\documentclass{dissertation}
\usepackage{microtype} 
\usepackage{mathpartir}
\usepackage{makecell}
\usepackage{tabularx}
\usepackage[export]{adjustbox}
\usepackage{pdfpages}
\usepackage{style/notation}
\usepackage{crossreftools}
\usepackage[natbib=true,style=alphabetic]{biblatex}
\usepackage{enumitem}
\usepackage{style/mycatchfile}
\usepackage{style/mynameref}
\usepackage[nopostdot]{glossaries}
\usepackage{multicol}

\usepackage{style/mypart}

\usetikzlibrary{positioning,fit,shapes,matrix}

\newif\ifthesis
\thesistrue

\newglossaryentry{BI}{
    name={\text{BI}},
    description={The logic of bunched implications.},
    sort={BI}
}
\newglossaryentry{APCP}{
    name={\texttt{APCP}},
    description={Asynchronous priority-based classical processes (cf.\ \gls{PCP}).},
    sort={APCP}
}
\newglossaryentry{piBI}{
    name={\texorpdfstring{$\pi\mkern-2mu$\texttt{BI}}{piBI}},
    description={A session-typed \picalc derived from \gls{BI}.},
    sort={piBI}
}
\newglossaryentry{alcalc}{
    name={\texorpdfstring{$\alpha\mkern-4mu\lambda$}{alpha-lambda}-calculus},
    description={A typed \lamcalc derived from \gls{BI}.},
    sort={alphalambdacalculus}
}
\newglossaryentry{MPST}{
    name={\text{MPST}},
    plural={MPSTs},
    description={Multiparty session type.},
    sort={MPST}
}
\newglossaryentry{picalc}{
    name={\texorpdfstring{$\pi$}{pi}-calculus},
    description={A mathematical model of message-passing.},
    sort={picalculus}
}
\newglossaryentry{basepi}{
    name={{\scriptsize{BASE}}\texorpdfstring{$\pi$}{-pi}},
    description={A basic variant of the \picalc from which the other variants in this thesis can be constructed.},
    sort={basepi}
}
\newglossaryentry{CP}{
    name={\texttt{CP}},
    description={Classical processes.},
    sort={CP}
}
\newglossaryentry{PCP}{
    name={\texttt{PCP}},
    description={Priority-based classical processes (cf.\ \CP).},
    sort={PCP}
}
\newglossaryentry{clpi}{
    name={\sff{s}$\pi^{\mkern-2mu+}$},
    description={A variant of the session-typed \picalc with a new operator for non-deterministic choice.},
    sort={spiplus}
}
\newglossaryentry{lamcalc}{
    name={\texorpdfstring{$\lambda$}{lambda}-calculus},
    plural={\texorpdfstring{$\lambda$}{lambda}-calculi},
    description={A mathematical model of functional programming.},
    sort={lambdacalculus}
}
\newglossaryentry{LAST}{
    name={\texttt{LAST}},
    description={Linear Asynchronous Session Types; an extension of the \lamcalc with session-typed message-passing (asynchronous, cyclic connections, may contain deadlocks).},
    sort={LAST}
}
\newglossaryentry{LASTs}{
    name={\texttt{LAST}\textsuperscript{$\star$}},
    description={Mild variant of \LAST.},
    sort={LASTs}
}
\newglossaryentry{LASTn}{
    name={\texttt{LAST}\textsuperscript{$\mathit{n}$}},
    description={Call-by-name variant of \LAST.},
    sort={LASTn}
}
\newglossaryentry{GV}{
    name={\texttt{GV}},
    description={Good variation; a variant of \LAST (synchronous, acyclic connections, deadlock-free).},
    sort={GV}
}
\newglossaryentry{PGV}{
    name={\texttt{PGV}},
    description={Priority-based good variation; a variant of \GV (synchronous, cyclic connections, deadlock-free).},
    sort={PGV}
}
\newglossaryentry{CGV}{
    name={\texttt{CGV}},
    description={Concurrent good variation; a variant of \GV with highly concurrent semantics (asynchronous, cyclically connected, deadlock-free).},
    sort={CGV}
}
\newglossaryentry{FSM}{
    name={\text{FSM}},
    plural={FSMs},
    description={Finite state machine.},
    sort={FSM}
}
\newglossaryentry{LTS}{
    name={\text{LTS}},
    plural={LTSs},
    description={Labeled transition system.},
    sort={LTS}
}
\newglossaryentry{piDILL}{
    name={$\pi\mkern-2mu$\texttt{DILL}},
    description={A session-typed pi-calculus derived from \gls{DILL}.},
    sort={piDILL}
}
\newglossaryentry{LL}{
    name={\text{LL}},
    description={Linear logic.},
    sort={LL}
}
\newglossaryentry{DILL}{
    name={\text{DILL}},
    description={Dual intuitionistic linear logic.},
    sort={DILL}
}
\newglossaryentry{MSC}{
    name={\text{MSC}},
    plural={MSCs},
    description={Message sequence chart.},
    sort={MSC}
}
\newglossaryentry{HMSC}{
    name={\text{HMSC}},
    plural={HMSCs},
    description={High-level \gls{MSC}.},
    sort={HMSC}
}
\newglossaryentry{CFSM}{
    name={\text{CFSM}},
    plural={CFSMs},
    description={Communication \gls{FSM}.},
    sort={CFSM}
}

\makeglossaries

\newglossarystyle{mylongstyle}{%
  \setglossarystyle{long}%
  {%
    \begin{longtable}{@{}p{0.2\textwidth}p{0.8\textwidth}@{}}%
  }%
  {\end{longtable}}%
  \renewcommand*{\glsgroupheading}[1]{}%
}

\begin{document}

\title{{Correctly Communicating Software: Distributed, Asynchronous, and Beyond}}
\author{Bas}{van den Heuvel}

\frontmatter

\newif\ifprintversion\printversionfalse
\makeatletter\if@print\printversiontrue\fi\makeatother
\newif\ifappendix\appendixtrue
\begin{titlepage}
	\ifprintversion\else
	\thispagestyle{empty}
	\begin{center}
		\ 
		
		\vspace{6em} \noindent {\huge%
			Correctly Communicating Software:
			\\
			Distributed, Asynchronous, and Beyond
		}

		\ifappendix
		\vspace{2em} \noindent {\LARGE%
			(extended version)
		}
		\fi
		
		\vspace{6em} \noindent {\Large \scshape%
			Bas van den Heuvel
		}
		
		\vfill \noindent {%
			Amsterdam \& Groningen, The Netherlands, 2024
		}
	\end{center}
	
	\newpage
	\fi
	\thispagestyle{empty}
	\vspace*{\fill} \noindent {%
		\setlength{\extrarowheight}{2em}
		\begin{tabularx}{\textwidth}{>{\raggedright\arraybackslash}p{.4\linewidth} X}
			\includegraphics[width=\linewidth,valign=t]{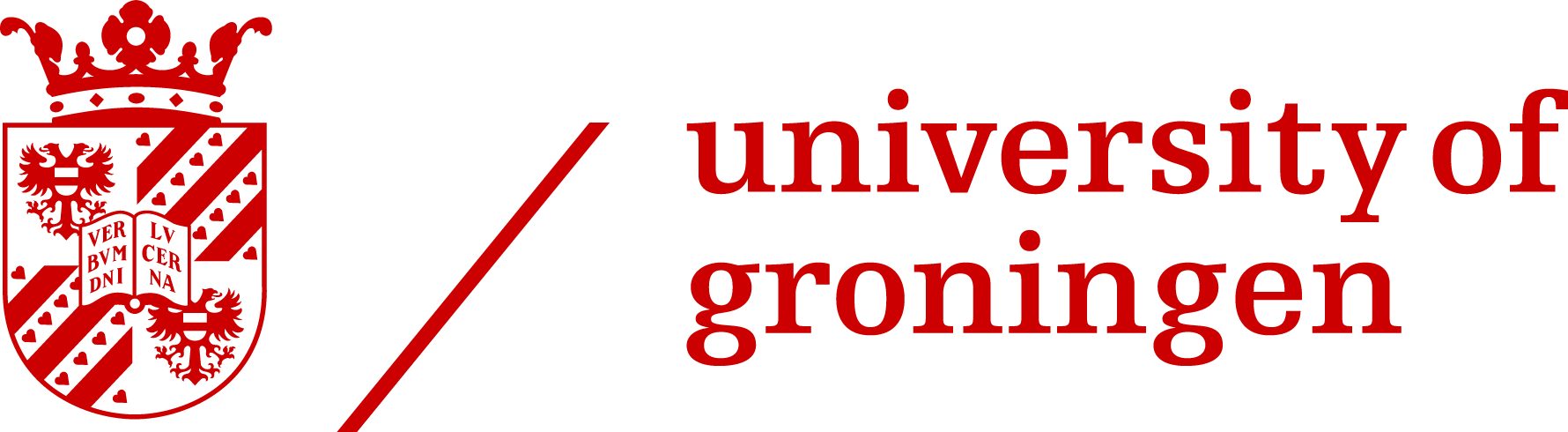}
			&
			The work in this book has been carried out at University of Groningen as part of the doctoral studies of the author.
			\\
			\includegraphics[width=.6\linewidth,valign=t]{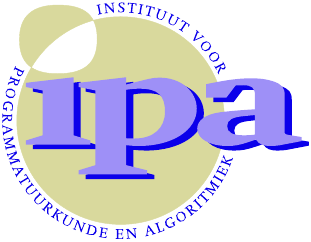}
			&
			{\bfseries IPA Dissertation Series No.\ 2024-03}
			\newline
			The work in this thesis has been carried out under the auspices of the research school IPA (Institute for Programming Research and Algorithms).
			\\
			\includegraphics[width=\linewidth,valign=t]{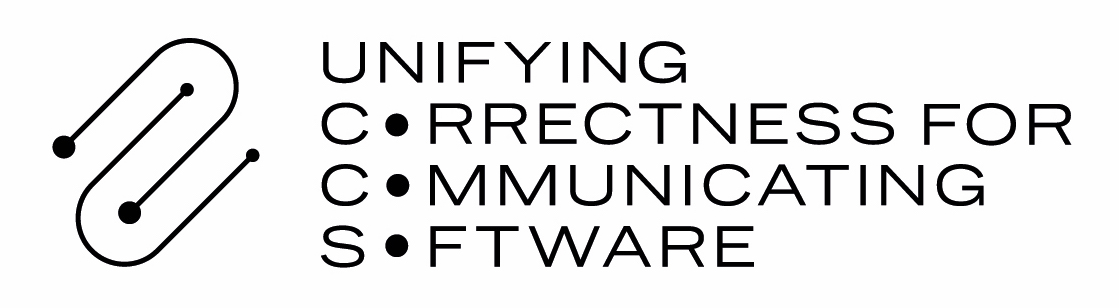}
			&
			The work in this thesis has been supported by the Dutch Research Council (NWO) under project No.\ 016.Vidi.189.046 (Unifying Correctness for Communicating Software).

		\end{tabularx}
	}

	\vfill \noindent {%
		Typeset with \LaTeX
		\\
		Printed by Ridder Print, The Netherlands
		\\
		Cover designed by Tijmen Zwaan
		\\
		Copyright \copyright\ Bas van den Heuvel, 2024
	}
	
	\clearpage
	\includepdf[pages=-]{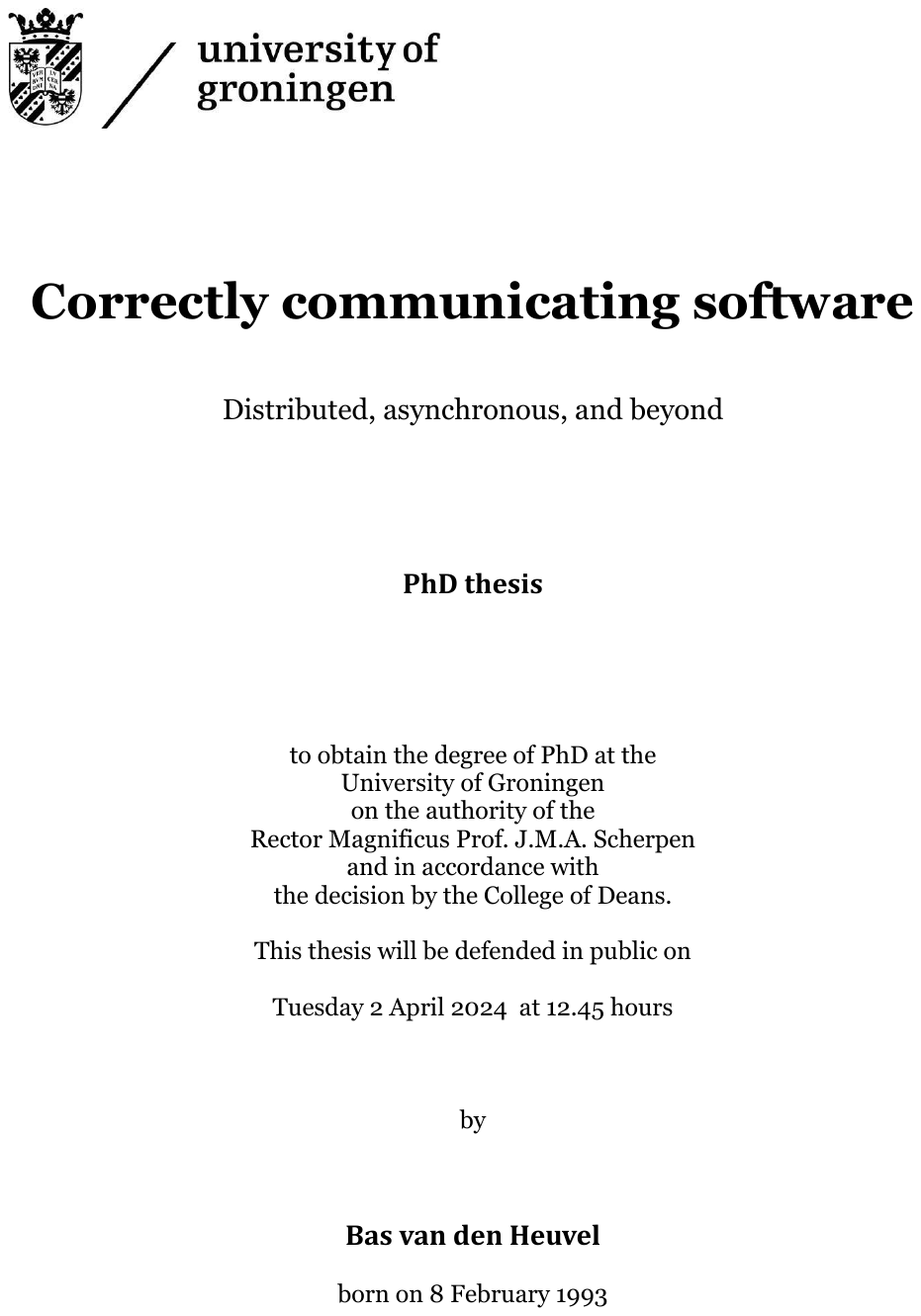}
\end{titlepage}


\chapter{Abstract}
Much of the software we use in everyday life consists of \emph{distributed} components (running on separate cores or even computers) that collaborate through \emph{communication} (by exchanging messages).
It is crucial to develop robust methods that can give reliable guarantees about the behavior of such message-passing software.
With a focus on \emph{session types} as \emph{communication protocols} and their foundations in logic, this thesis revolves around the following question:

\begin{quote}
    \emph{
        How can we push the boundaries of the logical foundations of session types (binary and multiparty), extending their expressiveness and applicability, while preserving fundamental correctness properties?
    }
\end{quote}

\noindent
In this context, this thesis studies several intertwined aspects of message-passing.
\begin{itemize}

    \item
        \textbf{Network Topology and Asynchronous Communication.}
        Deadlock-freedom is a notoriously hard problem, but session types can be enhanced to restrict component connections to guarantee deadlock-freedom.
        This thesis studies the effects of \emph{asynchronous communication} (non-simultaneous sending/receiving) on deadlock-freedom.

    \item
        \textbf{Non-determinism.}
        Non-deterministic choices model how programs may follow different paths of execution (e.g., under external influence).
        This thesis explores and compares non-de\-ter\-mi\-nis\-tic choice constructs with fine-grained semantics, with enhanced session types that maintain the expected correctness properties.

    \item
        \textbf{Ownership.}
        Session types have a foundation in \emph{linear logic}, a logic with fine-grained resource control in terms of \emph{number of uses}.
        This thesis develops an alternative logical foundation and accompanying variant of the pi-calculus based on \emph{the logic of bunched implications}, similar to linear logic but with a focus on \emph{ownership} of resources.

    \item
        \textbf{Functions.}
        Thus far, all models of message-passing are concurrent and process-based, but there are also such models for sequential, functional programming.
        This thesis presents several functional models of message-passing, and guarantees correctness through translation into the process-based models discussed above.

    \item
        \textbf{Multiparty Session Types (MPSTs): Asynchronous and Distributed.}
        MPSTs describe protocols between multiple components; they are practical, but correctness guarantees are complex.
        This thesis analyzes MPSTs implemented as distributed, asynchronously communicating networks by reduction to the binary session types discussed so far.

        Knowing exactly how a program is implemented (as assumed in the topics above) is not always practically feasible.
        This thesis adapts the approach above to use MPSTs to \emph{monitor} the behavior of programs with unknown specification, maintaining some correctness guarantees for asynchronous and distributed implementations of MPSTs.
\end{itemize}

\chapter{Summary for Non-experts}

We rely on software in every aspect of our everyday lives, so it is crucial that this software functions reliably.
The majority of software consists of smaller pieces of software that cooperate by communicating.
Reliable communication is therefore imperative to the reliability of software.
Ideally, software is guaranteed to function correctly by construction;
this contrasts common practice, where software is (non-exhaustively) tested for potential problems, or problems are even encountered only when it is in use.

This dissertation develops techniques that ensure that software communicates correctly by construction.
The main component herein is \emph{communication protocols} that describe precisely what is expected of the communication of pieces of software.
These protocols are then used to guide the development of correctly communicating software.

The approach is mathematical, taking heavy inspiration from logical reasoning.
Logic provides precise techniques to reason about the communication resources of software, that make guaranteeing correctness straightforward.
However, this straightforwardness comes with limitations: the techniques are restricted in the communication patterns they support.
An example restriction is that when more than two pieces of software need to communicate, they can only do so through a single point of connection.
In practice, pieces of software can communicate with each other directly, without such limitations.

Hence, the theme of this dissertation is \emph{pushing the boundaries of logic reasoning for communicating software}.
The goal is getting closer to reliable communication by construction for realistic software.
As such, the thesis develops new and improves existing techniques based on logic, while maintaining correctness guarantees.

\chapter{Samenvatting voor Non-experts}

In elk aspect van ons dagelijks leven leunen wij op software, dus is het cruciaal dat deze software betrouwbaar is.
Het merendeel van software bestaat uit kleinere stukken die samenwerken door te communiceren.
Betrouwbare communicatie is daarom onmisbaar voor het correct functioneren van software.
Idealiter wordt de correcte werking van software gegarandeerd door diens constructie;
echter, in de praktijk wordt software (onvolledig) getest op mogelijke problemen of worden problemen zelfs pas ontdekt tijdens gebruik van de software.

In deze dissertatie worden technieken ontwikkeld die verzekeren dat software correct communiceert vanuit constructie.
Het hoofdbestandsdeel hierin is \emph{communicatieprotocollen} die precies beschrijven wat er wordt verwacht van de communicatie van stukken software.
Deze protocollen worden vervolgens gebruikt als leidraad in de ontwikkeling van software die correct communiceert.

De aanpak is wiskundig en zwaar geïnspireerd door logisch redeneren.
Logica bevat precieze technieken voor het redeneren over communicatiemiddelen van software, die het garanderen van correctheid vereenvoudigen.
Deze eenvoud komt echter met tekortkomingen: de technieken zijn gelimiteerd in de ondersteunde communicatiepatronen.
Een voorbeeld is de beperking dat wanneer meer dan twee stukken software moeten communiceren, zij dit alleen kunnen doen via een enkel verbindingspunt.
In de praktijk kunnen stukken software direct met elkaar communiceren, zonder zulke beperkingen.

Derhalve is het thema van deze dissertatie het \emph{verleggen van de grenzen van logisch redeneren voor communicerende software}.
Het doel is dichter bij betrouwbare communicatie vanuit constructie voor realistische software komen.
Hiertoe ontwikkelt en verbetert deze dissertatie nieuwe en bestaande technieken gebaseerd op logica, terwijl correctheidsgaranties worden behouden.

\tableofcontents

\mainmatter

\thumbtrue

\def\fileLMCS{repos/apcp/submissions/LMCS/main.tex}
\newif\ifAPCP\APCPfalse
\newif\ifLASTn\LASTnfalse
\def\fileAPLAS{repos/s-pi-plus/APLAS23/main.tex}
\def\fileAPLASAppPi{repos/s-pi-plus/APLAS23/appendix/pi-proofs.tex}
\def\fileAPLASAppBisim{repos/s-pi-plus/APLAS23/appendix/pi-comparison.tex}
\def\fileLICSAppPi{repos/s-pi-plus/LICS23/appendix/pi-proofs.tex}
\newif\ifaplas\aplasfalse
\def\dirOOPSLA{repos/pi-bi/paper}
\def\fileOOPSLA#1{\dirOOPSLA/#1.tex}
\newif\ifpiBI\piBIfalse
\newif\ifalphalambda\alphalambdafalse
\def\fileSCICO{repos/apcp/submissions/SCICO-final/journal21.tex}
\newif\ifmpstAPCP\mpstAPCPfalse
\def\fileRV{repos/apcp/submissions/RVPaper/RV23/main.tex}
\def\fileRVApp{repos/apcp/submissions/RVPaper/RV23/appendix.tex}
\newif\ifrelProj\relProjfalse
\newif\ifmpstMon\mpstMonfalse

\chapter{Introduction}
\label{c:introduction}

Let me start off bluntly, by giving the shortest summary of my dissertation I can possibly give: it is about communicating software, how such communication can be described by protocols, and how we can use such protocols to guarantee that communicating software behaves as intended.
In this introduction, I unpack the components of the summary, to provide a gentle but detailed overview of my contributions.

In \Cref{s:intro:context}, I put this rather terse summary of my thesis gently into context, ending the section with a research question to which my dissertation gives a (partial) answer.
Then, in \Cref{s:intro:contrib}, I give an overview of the methods used to answer this question.
I then go through the topics addressed in this dissertation, motivating each topic with context and identifying precise research questions.
Finally, in \Cref{s:intro:outline}, I summarize the outline of the parts and chapters in this thesis, along with references to the publications that are derived from the research presented here.

\section{Context}
\label{s:intro:context}

It is very hard to imagine our lives without computers.
Are there even aspects of human life on Earth that are not touched by computers at all?
Even if you are not reading this dissertation on a computer, I can assure you that several computers were involved in putting it together.

Clearly, it is of utmost importance that computers are reliable.
You would appreciate your microwave not to overheat, or your digital life-support to never skip a heart beat.
But how can we be sure that computers are indeed reliable?
This is a complex task, with many aspects.

\paragraph{Reliable software.}

As of the day of writing, most to all software is made by humans.
Software architects, engineers, and programmers write and debug countless lines of code to create the firmware, operating systems, and applications we use in our everyday lifes.
It is no surprise (or shame!) that these people make the occasional error, introducing bugs into software eventually leading to unpleasant user experiences and even fatal crashes.

Finding bugs in the source code of software is thus an essential step in making software reliable.
In practice, many bugs are found by thorough testing, i.e., by running software under predetermined circumstances and checking that the result is as expected.
However, testing is by no means a guarantee for reliability: it is impossible to test all possible ways software can fail.
A more thorough approach is to exhaustively analyze a program's source code and determine whether all possible paths of execution will lead to an expected outcome; a salient approach herein is by means of \emph{types}.

\paragraph{Types.}

Commonly, programs are sequences of \emph{operations} on \emph{data}.
For these programs to function well, such operations make certain assumptions about the data.
For example, when an operation is to divide by two, it expects its input to be a number and not a string; the behavior of the operation on anything else than a number would be undefined and lead to unexpected outcomes.

To ensure that the outcome of such programs is as expected, we can assign types to data and operations.
Let me briefly illustrate the key idea of such types.
Suppose we are writing a program that takes as input a string, and is expected to return an integer representing the length of the string.
The program stores the input as a variable~\texttt{input}, and the program has access to an operation called~\texttt{len}; the program thus returns $\mathtt{len}(\mathtt{input})$.
In our program, we expect the input to be a string, so we assign it the corresponding type~$\mathtt{input} : \sff{str}$.
The \texttt{len}-operation expects as input a string, in which case it returns an integer; its type is thus~$\mathtt{len} : \sff{str} \to \sff{int}$.
Thus, the type of the operation applied to our variables results in an integer~$\mathtt{len}(\mathtt{input}) : \sff{int}$.
Equiped with these types, a compiler or execution environment can warn or abort when our program is given anything else than a string as input, and guarantee that the output is an integer when it is given a string.

Types are a main theme in my thesis, though as we will see, I will not be focusing on types for data and operations.
This is because in this dissertation, the focus is on software that employs \emph{concurrency}.

\paragraph{Concurrency as a way to better exploit modern hardware.}

Informally, concurrency is the principle that things can happen in arbitrary order, or even simultaneously.
There are many ways to employ concurrency in software.
The most well-known application of concurrency is when a computer has a processor with multiple cores: independent parts of a program can run concurrently (in parallel) on different cores; this way, programs can be perform their computations significantly faster than if they were executed sequentially.
A more ubiquitous though less obvious application of concurrency is in \emph{distributed systems}, which provide a context for the better part of this thesis.

\paragraph{Distributed systems.}

Distributed systems consist of multiple programs that do not necessarily run on the same computer but work together to achieve a task.
Many of the applications on your smartphone work in this way: they employ several services that run in ``the cloud'' to perform small tasks.
For example, the widget you may have on your home screen does not measure the weather and collect the news itself: it requests this information from weather and news services.

In distributed systems, programs collaborate by exchanging messages.
Such messages contain, e.g., requests to perform a specific task, or their results.
This form of concurrency is therefore referred to as \emph{message-passing concurrency}; it is the main subject of this thesis.

\paragraph{Message-passing concurrency.}

In message-passing concurrency, programs collaborate by exchanging messages.
Such messages can carry data, for example a user's login credentials.
Messages can also be used to coordinate on choices: for example, one program chooses an operation and the other programs should respond accordingly.

What contributes to reliability in message-passing concurrency?
There are many aspects.
On the network level, reliability means that every message reaches its destination unaltered and that programs do not change location unannounced or disappear altogether.
On the program level, reliability means that message-exchange is safe (e.g., programs do not try to simultaneously send each other messages) and that programs respond to the messages they receive as expected.

In this thesis, I abstract away from concerns about network reliability, and focus on the reliability of message-passing in programs.
In particular, the focus is on how the exchange of messages can be captured by \emph{communication protocols} and how these protocols can help to guarantee the reliability of message-passing programs.

\paragraph{Message-passing protocols.}

Communication protocols describe sequences of message exchanges.
Such exchanges usually describe the kind of data exchanged, but also moments of choice: protocol branches where the message determines how to proceed.

These protocols can be used to describe interactions between message-passing programs, but also---and perhaps more useful---to \emph{prescribe} a behavior to programs.
That is, the protocols declare how a program is expected to behave, and so the protocol can be used to \emph{verify} the behavior of a program.

To describe more clearly the relation between protocols and verification, let me elaborate on what I mean by the behavior of message-passing programs.
In message-passing, a program runs in parallel with other programs.
The program is connected to other programs via channels, over which messages can be exchanged.
The message-passing behavior of this program then comprises how and when it sends and receives messages over these channels.
A communication protocol can then prescribe message-passing behavior to a program, be it for one specific channel or for all channels at once.

This way, communication protocols are used as \emph{behavioral types} for message-passing programs.
This is a widely used approach to verification in message-passing concurrency.
An especially prominent branch of research focuses on behavioral types called \emph{session types}.

\paragraph{Session types.}

Originally introduced by Honda \etal~\cite{conf/concur/Honda93,conf/esop/HondaVK98}, session types denote sequences of messages and branches precribing message-passing behavior to programs.
Because they describe protocols for specific channels, so between two programs, they are often referred to as \emph{binary session types}.
Later, I discuss the generalization of binary session types as multiparty session types.

Let us consider an example protocol, which prescribes a program to offer a choice between two branches.
In the first branch, labeled \sff{login}, the program should receive an integer and then make a choice between two alternatives.
If the program chooses \sff{success}, the protocol ends; if the program chooses \sff{failure}, the protocol starts over from the beginning.
In the second branch, labeled \sff{quit}, the protocol simply ends.
We can write this protocol formally as the following session type:
\begin{align}
    S \deq \mu X . \& \big\{ \sff{login} . {?}(\sff{int}) . \oplus \{ \sff{success} . \tEnd , \sff{failure} . X \} , \sff{quit} . \tEnd \big\}
    \label{eq:intro:st}
\end{align}
As the protocol may repeat, it starts with a recursive definition $\mu X \ldots$, and contains a recursive call $X$ that indicates a jump to the corresponding recursive definition.
The branch is denoted $\& \{ \sff{login} \ldots , \sff{quit} \ldots \}$, where the labels prefix session types that prescribe behavior for each branch.
Reception is denoted ${?}(T)$, where $T$ is the type of the value to be received.
To prescribe a choice, we use selection, denoted $\oplus \{ \ldots \}$; again, the branches contain labels that prefix corresponding session types.
Whenever the protocol should end, we use $\tEnd$.

Session types are very useful for verifying the behavior of message-passing programs.
Key to the power of session types for verification is the notion of \emph{duality}.
Duality says that if a program at one end of a channel behaves according to a given binary session type, then the program at the other end of the channel should show a complementary behavior: if the original type describes a send, the dual describes a receive, and so forth.

Let us consider the dual of $S$~\eqref{eq:intro:st}.
This protocol prescribes a program to choose between two alternatives.
If the program chooses \sff{login}, it should send an integer and offer a choice between two alternatives: case \sff{success} ends the protocol, and case \sff{failure} repeats it; if the program chooses \sff{quit}, the protocol ends.
Indeed, this protocol is complementary to that described by~$S$.
Formally, we can write the dual of~$S$, denoted $\ol{S}$, as follows:
\[
    \ol{S} = \mu X . \oplus \{ \sff{login} . {!}(\sff{int}) . \& \{ \sff{success} . \tEnd , \sff{failure} . X \} , \sff{quit} . \tEnd \}
\]
The only new notation here is ${!}(T)$, which denotes sending a value of type $T$.

Duality is the cornerstone of session type theory, as it is instrumental in proving that \emph{communication is safe} in message-passing programs.
Communication safety is one of the important correctness properties of message-passing concurrency, discussed next.

\paragraph{Correctness properties for message-passing.}

In general, we want to verify that programs are \emph{correct}.
But what is correct?
What are the correctness properties we are after?
In sequential programming, correctness focuses on the input-output behavior of the program (i.e., given a certain input, the output of the program is as expected), but this does not necessarily apply to message-passing programs.

The main correctness property of message-passing programs is \emph{protocol conformance}, often referred to also as \emph{protocol fidelity}.
It means that, given a protocol (e.g., a binary session type) for a specific channel, a message-passing program correctly implements the sequence of communications specified by the protocol.

The next major correctness property is \emph{communication safety}.
It means that programs at either end of a channel do not exhibit conflicting communications on that channel, i.e., communication errors.
For example, if both programs simultaneously send something, then the messages will ``run into each other'', leading to unexpected behavior.
Communication safety then signifies the absence of communication errors.
This property often follows directly from using session types, relying on duality as explained above.

Another important property is \emph{deadlock-freedom}.
A deadlock occurs when connected programs are waiting for each other, without a possibility to resolve.
For example, consider three programs $P$, $Q$, and $R$.
Program $P$ is ready to send something to $Q$, while $Q$ is ready to send something to $R$, and $R$ is ready to send something to $P$.
Under the assumption that communication is synchronous, i.e., that a send can only take place once the receiving program is ready to receive (I will elaborate on this later), this situation signifies a deadlock: all the programs are stuck waiting for each other.
Deadlock-freedom then means that situations leading to deadlocks never occur.
This property is a major challenge in this thesis, much more so than protocol fidelity and communication safety.
As such, I will elaborate further on it later in this introduction.

\paragraph{Logical foundations for session-typed message-passing.}

It is a challenging task to design (session) type systems for message-passing that guarantee correctness properties by typing directly.
It is then only natural to look to existing solutions to similar problems in different areas as inspiration.
Logic is an important such inspiration for correctness in programming.

As observed by Curry~\cite{journal/pnas/Curry34} and Howard~\cite{journal/HBCurry/Howard80}, intuitionistic logic serves as a deep logical foundation for sequential programming.
To be precise, the logic is isomorphic to the simply typed \lamcalc: propositions as types, natural deduction as type inference, and proof normalization as $\beta$-reduction.
As such, properties of intuitionistic logic transfer to the simply typed \lamcalc.
For example, from the strong normalization property of intuitionistic logic we know that well-typed \lamcalc terms are strongly normalizing as well, i.e., every well-typed term can be executed to reach a final state where there is no more computation to be done (see, e.g., \cite{book/Girard89}).

Linear logic, introduced by Girard~\cite{journal/tcs/Girard87}, is a logic of resources.
That is, linear logic features a precise management of resources governed by the principle of \emph{linearity}: each resources must be used exactly once.
This means that, unlike in common logics such as classical and intuitionistic logic, linear logic does not allow duplication (contraction) and discarding (weakening) of resources.

Many thought that a correspondence in the style of Curry-Howard could be found between linear logic and concurrent programming, though early attempts were unsatisfactory (e.g., by Abramsky~\cite{journal/tcs/Abramsky93}).
It took over two decades since the conception of linear logic for a satisfactory answer to be found: Caires and Pfenning~\cite{conf/concur/CairesP10} discovered that linear logic is very suitable as a logical foundation for the session-typed \picalc.
They found the isomorphism in dual intuitionistic linear logic~\cite{report/BarberP96}: propositions serve as session types, sequent calculus derivations as typing inferences, and cut-reduction as communication.
Not much later, Wadler~\cite{conf/icfp/Wadler12,journal/jfp/Wadler14} showed a similar correspondence with the original classical linear logic.

The works by Caires and Pfenning and by Wadler are generally accepted as canonical Curry-Howard interpretations of linear logic, and have inspired a magnitude of follow-up work.
The main reason is that session type systems derived from linear logic precisely identify message-passing programs that are deadlock-free.
We will explore this in more depth later.

\subsection{Encompassing Research Question}
\label{s:intro:context:requ}

Hopefully, at this point I have managed to convince you of the ubiquity of communicating software and the need for guaranteeing its correctness.
I have introduced a theoretical approach to this problem by considering abstractions of communicating software as message-passing programs, and to capturing their interactions as communication protocols called session types.
This dissertation includes a broad range of topics, each approaching the correctness of message-passing programs using session types from different angles.

Nonetheless, every topic has a common starting point: logical foundations for session-typed message-passing.
As discussed above, session typing has deep foundations in linear logic.
These foundational connections induce precise patterns of message-passing for which the correctness properties discussed above hold straightforwardly.
Yet, there are many patterns of message-passing outside of these logical foundations that still enjoy those correctness properties.
As will be a theme in this thesis, these logical foundations do not account for, e.g., practical aspects of message-passing such as asynchronous communication, resource management, and the way programs can be connected.

Hence, the following research question covers the encompassing theme of my thesis:

\begin{restatable}{encrequ}{theEncrequ}
    How can we push the boundaries of the logical foundations of session types (binary and multiparty), extending their expressiveness and applicability, while preserving fundamental correctness properties?
\end{restatable}

\begin{figure}[t]
    \begingroup
    \setlength{\topsep}{0pt}
    \setlength{\partopsep}{0pt}
    \begin{restatable}{requI}{requAPCP}
        \label{requ:APCP}
        How can we reconcile asynchronous communication, deadlock-freedom for cyclically connected processes, and linear logic foundations for session types?
    \end{restatable}

    \begin{restatable}{requI}{requND}
        \label{requ:ND}
        How can we increase commitment in non-deterministic choices, while retaining correctness properties induced by linear logic foundations for session types?
    \end{restatable}

    \begin{restatable}{requI}{requBI}
        \label{requ:BI}
        Can \BI serve as an alternative logical foundation for session types, and if so, what correctness guarantees does this bring?
    \end{restatable}
    \endgroup

    \smallskip
    \begingroup
    \setlength{\topsep}{0pt}
    \setlength{\partopsep}{0pt}
    \begin{restatable}{requII}{requLASTn}
        \label{requ:LASTn}
        Can we exploit \APCP's asynchronous message-passing and cyclic connections in sequential programming, while retaining correctness properties?
    \end{restatable}

    \begin{restatable}{requII}{requAL}
        \label{requ:AL}
        Can \pibi faithfully represent the \alamcalc on a low level of abstraction, and can we exploit \pibi's correctness properties in the \alamcalc?
    \end{restatable}
    \endgroup

    \smallskip
    \begingroup
    \setlength{\topsep}{0pt}
    \setlength{\partopsep}{0pt}
    \begin{restatable}{requIII}{requRelProj}
        \label{requ:relProj}
        Can we define a method for obtaining local perspectives from global types that induces a class of well-formed global types featuring communication patterns not supported before?
    \end{restatable}

    \begin{restatable}{requIII}{requMpstAPCP}
        \label{requ:mpstAPCP}
        Can we leverage on binary session types to guarantee protocol conformance and deadlock-freedom for distributed (model) implementations of \MPSTs?
    \end{restatable}

    \begin{restatable}{requIII}{requMpstMon}
        \label{requ:mpstMon}
        Can we dynamically verify distributed blackbox implementations of \MPSTs, and what correctness properties can we guarantee therein?
    \end{restatable}
    \endgroup

    \caption{Research questions.}
    \label{f:requs}
\end{figure}

In \Cref{f:requs}, I break this down into eight more precise research questions.
Being precise, these research questions contain terminology not yet introduced.
The next section discusses the contributions presented in my dissertation, and how they address the research questions in \Cref{f:requs}.
In this section, all the terminology becomes clear.

\section{Contributions}
\label{s:intro:contrib}

Thus far, I have (intentionally) been vague about what constitutes a ``program''.
There are many ways to define a program, but ultimately it boils down to context: for what purpose are you working with programs?
In my dissertation, I work with two specific notions of message-passing programs.

In \Cref{s:intro:contrib:pi}, I shall introduce the process-based approach.
In this approach, processes interact through message-passing without details about internal computation, i.e., processes are abstract representations of programs by focusing solely on their message-passing behavior.
Here, I introduce \Cref{c:APCP,c:clpi,c:piBI} in \Cref{p:pi}, which build on this process-based approach and address \Cref{requ:APCP,requ:ND,requ:BI} in \Cref{f:requs}, respectively.

In \Cref{s:intro:contrib:lambda}, I introduce a functional approach.
In this approach, we consider functional programs whose behavior is precisely defined, and consider how we can use message-passing to derive a deeper understanding of their behavior.
\Cref{c:LASTn,c:alphalambda} build on this approach in \Cref{p:lambda} and address \Cref{requ:LASTn,requ:AL} in \Cref{f:requs}, respectively.

In \Cref{s:intro:contrib:mpst}, I introduce \Cref{p:mpst}, which contains chapters on multiparty session types.
These \Cref{c:relProj,c:mpstAPCP,c:mpstMon} address \Cref{requ:relProj,requ:mpstAPCP,requ:mpstMon} in \Cref{f:requs}, respectively.

\subsection{Binary Session Types and Message-passing Processes (\crtCref{p:pi})}
\label{s:intro:contrib:pi}

The context I have introduced is that of message-passing concurrency, and the usage of session types as communication protocols to verify several correctness aspects therein.
It thus makes sense to give a definition of program that focuses on message-passing.

As the research presented in this thesis is of theoretical nature, I will define \emph{mathematical models} of message-passing programs.
This gives us the tools to precisely define the behavior of such programs, and to formally prove properties of them.
Many such formalisms already exist; instead of re-inventing the wheel, I base my definitions on such prior works.

\paragraph{The \picalc.}

The mathematical model of message-passing which most of this thesis leans on is the~\picalc.
Originally introduced by Milner \etal~\cite{book/Milner89,journal/ic/MilnerPW92}, the \picalc is a \emph{process calculus} focused entirely on message-passing.
In the \picalc, processes run in parallel and communicate by exchanging messages on dedicated channels.

We consider an example of message exchange on a channel:
\begin{align}
    P_1 \deq \pRes{xy} ( \pSel* x < \sff{helloWorld} ; \0 \| \pBra* y > \sff{helloWorld} ; \0 )
    \label{eq:intro:P1}
\end{align}
The process $P_1$ contains two parallel subprocesses, separated by `$\,\|\,$'.
These subprocesses can communicate because they are connected by a channel, indicated by the restriction `$\pRes{xy}$' wrapping the parallel composition.
This way, the subprocesses have access to the channel's endpoints $x$ and $y$.
The left subprocess sends on $x$ a label \sff{helloWorld} (denoted~$\pSel* x < \ldots$).
The continuation of this send, indicated by `$\,;\,$', is $\0$, which indicates that the subprocess is done.
The right subprocess receives on $y$ the same label (denoted~$\pBra* y > \ldots$), and also continues with $\0$.
This way, the subprocesses communicate by exchanging the label \sff{helloWorld}.

The behavior of \picalc processes is then defined in terms of a \emph{reduction semantics} that defines how parallel subprocesses that do complementary message exchanges on connected endpoints synchronize.
In the case of $P_1$~\eqref{eq:intro:P1}, the send on $x$ and the receive on $y$ can synchronize.
Afterwards, both subprocesses transition to $\0$:
\[
    P_1 \redd \pRes{xy}( \0 \| \0 )
\]
Since the endpoints $x$ and $y$ are no longer used, the channel can be garbage collected.
Moreover, the two inactive subprocesses can be merged into one.
This clean-up is often done implicitly, resulting in the following reduction:
\begin{align}
    P_1 \redd \0
    \label{eq:intro:P1Redd}
\end{align}

\paragraph{Channel mobility.}

The~\picalc can express much more interesting forms of message-passing than exchanging labels.
The distinctive feature of the~\picalc is \emph{channel mobility}: the exchange of channel endpoints themselves, influencing the way processes are connected.

The following example is similar to $P_1$~\eqref{eq:intro:P1}.
However, before exchanging the message, the subprocesses exchange a channel, changing their connections.
\begin{align}
    P_2 \deq \pRes{xy} \big( \pRes{zw} ( \pOut* z[u] ; \pFwd [u<>x] \| \pIn w(v) ; \pSel* v < \sff{helloWorld} ; \0 ) \| \pBra* y > \sff{helloWorld} ; \0 \big)
    \label{eq:intro:P2}
\end{align}
Process $P_2$ consists of two subprocesses, connected on a channel with endpoints $x$ and $y$.
The right subprocess expects to receive a label \sff{helloWorld}.
However, the left subprocess is further divided into two parallel subprocesses, connected on a channel with endpoints $z$ and~$w$.
In a nutshell, the left sub-subprocess has access to the endpoint $x$, but sends it to the right sub-subprocess which proceeds to send the label \sff{helloWorld} on it.
This is implemented in the left sub-subprocess by sending on $z$ a new endpoint $u$ (denoted~$\pOut* z[u]$), after which $u$ is linked to $x$ by means of a forwarder `$\pFwd [u<>x]$' that passes any messages incoming on $u$ to $x$ and vice versa.
In the right sub-subprocess an endpoint $v$ is received on $w$, after which the label is sent on $v$.

Unlike in $P_1$~\eqref{eq:intro:P1}, in $P_2$~\eqref{eq:intro:P2} we cannot synchronize on $x$ and $y$, as there is no send on~$x$ available.
We first need to resolve the communication on endpoints $z$ and $w$.
The left subprocess sends a new name $u$, so this creates a new channel; the endpoints $z$ and $w$ are no longer used, so this channel can be garbage collected:
\[
    P_2 \redd \pRes{xy} \big( \pRes{uv} ( \pFwd [u<>x] \| \pSel* v < \sff{helloWorld} ; \0 ) \| \pBra* y > \sff{helloWorld} ; \0 \big)
\]
Still, a synchronization on $x$ and $y$ is not possible.
Instead, we can eliminate the forwarder $\pFwd [u<>x]$.
The forwarder connects $x$ to $u$, and $u$ is connected to $v$ through a channel, so we can simply do all the message exchanges that were on $v$ before on $x$ directly:
\[
    \redd \pRes{xy} ( \pSel* x < \sff{helloWorld} ; \0 \| \pBra* y > \sff{helloWorld} ; \0 ) = P_1
\]
Hence, eventually we reach $P_1$~\eqref{eq:intro:P1}, which reduces as before~\eqref{eq:intro:P1Redd}.

\paragraph{Branching.}

Another important element of the~\picalc is branching: a process \emph{branches} on a choice between a set of labeled alternatives on one channel endpoint, and the process on the opposite endpoint \emph{selects} an alternative by sending one of the offered labels.
This way, the branching process can prepare different behaviors depending on the received label.

The following example illustrates how a process can branch to different behaviors:
\begin{align}
    P_3 &\deq \pBra* x > {\{ \sff{send}: \pOut* x[z] ; \0 , \sff{quit}: \0 \}}
    \label{eq:intro:P3}
\end{align}
Process~$P_3$ offers on $x$ a choice between behaviors labeled \sff{send} and \sff{quit}.
The behavior of the \sff{send}-branch is to send a fresh channel $z$ on $x$, while the behavior in the \sff{quit}-branch ends.

Now consider the following processes, that we intend to connect to $P_3$~\eqref{eq:intro:P3}.
Each process selects a different behavior, and thus they behave differently after the selection.
\begin{align*}
    Q_1 &\deq \pSel* y < \sff{send} ; y(w) ; \0
    \\
    Q_2 &\deq \pSel* y < \sff{quit} ; \0
\end{align*}
Process $Q_1$ selects on $y$ \sff{send} and receives on $y$, whereas $Q_2$ selects on $y$ \sff{quit} and ends.
This way, we can connect either process with $P_3$~\eqref{eq:intro:P3}, i.e., as $\pRes{xy} ( P_3 \| Q_1 )$ or as $\pRes{xy} ( P_3 \| Q_2 )$.

\paragraph{Session types and the \picalc.}

I have already anticipated the usefulness of session types as communication protocols in message-passing concurrency.
Indeed, session types are widely used to verify the behavior of \picalc processes on channels.
Given a channel (i.e., two connected endpoints), a sequence of communications on that channel is referred to as a \emph{session}.
The behavior of a process on an endpoint can thus be ascribed a session type.

Consider the following process, where subprocesses create and exchange a new channel to communicate on:
\[
    P_4 \deq \pRes{xy} ( \pOut* x[z] ; \pIn z(u) ; \0 \| \pIn y(w) ; \pOut* w[v] ; \0 )
\]
In process~$P_4$, the left subprocess sends on $x$ a fresh channel $z$, on which it receives.
Recall the notation of session types introduced on \Cpageref{eq:intro:st}.
Hence, in the left subprocess we can ascribe to $x$ the session type ${!}({?}\tEnd . \tEnd) . \tEnd$: on $x$ we send an endpoint typed ${?}\tEnd . \tEnd$.
In the right subprocess of $P_4$, we receive on $y$ a channel $w$ on which we send.
So, here we ascribe to $y$ the session type ${?}({!}\tEnd . \tEnd) . \tEnd$.
This session type is the dual of the type ascribed to $x$ in the left subprocess.
Hence, the left and right subprocesses describe complementary behaviors on $x$ and $y$, and thus their exchanges on $x$ and $y$ are \emph{safe}.

Session typing also prevents branching from leading to unexpected behavior.
This anticipates an aspect of branching that will be important in \Cref{s:intro:contrib:mpst}: what if a process that branches on $x$ also uses endpoints other than $x$ for communications?
For example, the following process branches on an endpoint $w$ to define different behaviors on another endpoint $y$:
\NewDocumentCommand\eqIntroPfive{s}{\ensuremath{P_5 \IfBooleanTF{#1}{=}{\deq} \pRes{xy} ( \pIn x(v) ; \0 \| \pBra* w > {\{ \sff{send}: \pOut* y[u] ; \0 , \sff{quit}: \0 \}} )}}
\begin{align}
    \eqIntroPfive
    \label{eq:intro:P5}
\end{align}
In process~$P_5$, endpoints $x$ and $y$ form a channel, and the left subprocess receives on $x$.
The right subprocess branches on another endpoint $w$.
In the \sff{send}-branch, it sends on $x$, but in the \sff{quit}-branch it immediately ends.
If the right subprocess receives on $w$ the label \sff{send}, the communication between $x$ and $y$ proceeds as expected.
However, if the right subprocess were to receive on $w$ the label \sff{quit}, the receive on $x$ in the left subprocess is left without a corresponding send: the exchanges on $x$ and $y$ are not guaranteed to be safe.

To guarantee safety for branching, session typing requires that a process' behavior on endpoints other than the branching endpoint has the same type across all branches.
This way, process~$P_5$~\eqref{eq:intro:P5} is not considered well-typed, avoiding unsafe branching: the behavior of the right subprocess on $y$ in the \sff{send}-branch is ${!}(\tEnd) . \tEnd$, but $\tEnd$ in the \sff{quit}-branch.

By assigning session types to the endpoints of processes, session typing guarantees protocol fidelity.
Requiring that connected endpoints should be typed dually and that endpoints should be typed consistently across branches is what guarantees communication safety.
However, as we will see, session types do not guarantee deadlock-freedom by themselves.

\paragraph{Session types and deadlock.}

Consider the following process, where two subprocesses communicate on two different channels:
\NewDocumentCommand\eqIntroPsix{s}{\ensuremath{P_6 \IfBooleanTF{#1}{=}{\deq} \pRes{xy} \pRes{zw} ( \pOut* x[u] ; \pIn z(v) ; \0 \| \pOut* w[v'] ; \pIn y(u') ; \0 )}}
\begin{align}
    \eqIntroPsix
    \label{eq:intro:P6}
\end{align}
This process connects endpoint $x$ to $y$ and $z$ to $w$, so these endpoints should be assigned pairwise dual session types.
Indeed, in the left subprocess $x$ is typed ${!}(\tEnd) . \tEnd$ and $z$ is typed ${?}(\tEnd) . \tEnd$, and in the right subprocess $y$ is typed ${?}(\tEnd) . \tEnd$ and $w$ is typed ${!}(\tEnd) . \tEnd$.
Hence, the communications in $V$ are deemed safe.

However, process~$P_6$~\eqref{eq:intro:P6} is clearly not deadlock-free: the send on $x$ has to wait for the corresponding receive on $y$, but it is blocked by the send on $w$ which has to wait for the receive on $z$, itself blocked by the send on $x$.
In other words, $P_6$ is well-typed but deadlocked.

The above is a canonical example of deadlock caused by circular dependencies: $x$ depends on $y$, $y$ depends on $w$, $w$ depends on $z$, $z$ depends on $x$, and so on.
In fact, at the level of abstraction of the~\picalc, circular dependencies are the only source of deadlock.
It thus seems that if we have a type system that excludes circular depenencies, then typing does guarantee deadlock-freedom.
A prominent approach is rooted in deep connections between session-typed message-passing concurrency and \emph{linear logic}~\cite{conf/concur/CairesP10,conf/icfp/Wadler12}.

\paragraph{Linear logic and deadlock-freedom.}

As hinted to before, linear logic is an important foundation for session type systems that guarantee deadlock-freedom by typing.
To see why, let us be a bit more formal regarding session typing for the \picalc.

When typing a process, we use \emph{typing judgments}, denoted $\vdash P \typInf \Gamma$.
Here $\Gamma$ is a \emph{typing context}: a list of \emph{typing assignments} $x:A$, denoting that endpoint $x$ is assigned the session type $A$.
This way, $\vdash P \typInf \Gamma$ says that the endpoints used in $P$ (that are not connected by restriction) are typed according to $\Gamma$.
The typing judgment of a process is obtained by means of typing inference, following typing rules.

Note that we use a non-standard notation for our typing judgments, as opposed to the more common notation $\Gamma \vdash P$. This is because this notation can be confusing for readers unfamiliar with session type systems: from a logic point of view as well as a traditional typing point of view, a judgment $\Gamma \vdash P$ can be interpreted to mean ``from the facts in $\Gamma$, we derive that $P$ holds,'' which is not in line with the intended meaning of the judgment.
Instead, we write $\vdash P \typInf \Gamma$ to be interpreted as ``we derive that $P$ is well-typed according to~$\Gamma$.''

Let us consider a pair of typing rules that are traditionally used for session types.
These are the rules for parallel composition and restriction (endpoint connection), respectively.
Remember that we write $\ol{A}$ to denote the dual of $A$.
\begin{mathpar}
    \begin{bussproof}[par]
        \bussAssume{
            \vdash P \typInf \Gamma
        }
        \bussAssume{
            \vdash Q \typInf \Delta
        }
        \bussBin{
            \vdash P \| Q \typInf \Gamma , \Delta
        }
    \end{bussproof}
    \and
    \begin{bussproof}[res]
        \bussAssume{
            \vdash P \typInf \Gamma , x:A , y:\ol{A}
        }
        \bussUn{
            \vdash \pRes{xy} P \typInf \Gamma
        }
    \end{bussproof}
\end{mathpar}
Thus, Rule~\ruleLabel{par} places two processes in parallel and combines their (disjoint) typing contexts.
Rule~\ruleLabel{res} requires two endpoints of dual types, and connects them thus removing them from the typing context.
These are the rules that are necessary to type the endpoint connections in process~$P_6$~\eqref{eq:intro:P6}, i.e., they do not exclude the circular dependencies that cause deadlock.

Now consider the following typing rule, inspired by linear logic's Rule~\ruleLabel{cut}:
\[
    \begin{bussproof}[cut]
        \bussAssume{
            \vdash P \typInf \Gamma, x:A
        }
        \bussAssume{
            \vdash Q \typInf \Delta, y:\ol{A}
        }
        \bussBin{
            \vdash \pRes{xy} ( P \| Q ) \typInf \Gamma, \Delta
        }
    \end{bussproof}
\]
This rule combines Rules~\ruleLabel{par} and~\ruleLabel{res} into one.
It places two processes in parallel, connecting them on \emph{exactly one} pair of endpoints.
Further connections between $P$ and $Q$ are not possible, as there are no rules to connect endpoints in $\Gamma,\Delta$.

Session type systems based on linear logic replace Rules~\ruleLabel{par} and~\ruleLabel{res} by Rule~\ruleLabel{cut}.
As a consequence, these type systems guarantee deadlock-freedom by typing.
In a nutshell, this holds because Rule~\ruleLabel{cut} is the only way to connect processes: it is impossible to cyclically connect processes thus ruling out circular dependencies.
This way, the deadlocked process~$P_6$~\eqref{eq:intro:P6} is not considered well-typed in these type systems.

\paragraph{A base presentation of the \picalc.}

The \picalc can take some effort to get used to.
Because it is an important backbone of this thesis, I include \Cref{c:basepi} to allow the reader to familiarize themselves with the \picalc and the kind of session types used in this thesis.
The chapter introduces a variant of the session-typed \picalc called \basepi, that is ``distraction-free'' (i.e., \basepi focuses on the core of message-passing and session types).
This \basepi can be seen as a basis for the other variants of the \picalc introduced in this dissertation.
The chapter explains how \basepi relates to those other variants.

\subsubsection{Deadlock-freedom in Cyclically Connected Processes}
\label{s:intro:contrib:pi:APCP}

The interpretation of linear logic's Rule~\ruleLabel{cut} guarantees deadlock-freedom by typing, with a fairly straightforward proof following linear logic's \emph{cut-reduction} principle.
However, this approach is an \emph{overestimation}: a large class of deadlock-free processes is not typable using Rule~\ruleLabel{cut}.

Consider the following process, which is an adaptation of process~$P_6$~\eqref{eq:intro:P6} where the send and receive in the right subprocess are swapped in order:
\NewDocumentCommand\eqIntroPseven{s}{\ensuremath{P_7 \IfBooleanTF{#1}{=}{\deq} \pRes{xy} \pRes{zw} ( \pOut* x[u] ; \pIn z(v) ; \0 \| \pIn y(u') ; \pOut* w[v'] ; \0 )}}
\begin{align}
    \eqIntroPseven
    \label{eq:intro:P7}
\end{align}
Just like $P_6$, process~$P_7$ is a cyclically connected process.
However, this time $P_7$ is not deadlocked: first the communication between $x$ and $y$ can take place, and then the communication between $z$ and $w$.
Yet, $P_7$ cannot be typed using Rule~\ruleLabel{cut}.
The question is then: how do we design a session type system that considers $P_7$ well-typed, while guaranteeing deadlock-freedom for all well-typed processes (excluding, e.g., $P_6$~\eqref{eq:intro:P6})?

Even before the discovery by Caires and Pfenning of using linear logic's Rule~\ruleLabel{cut} to guarantee deadlock-freedom, Kobayashi~\cite{conf/concur/Kobayashi06} introduced a session type system that answers precisely this question.
In his approach, session types are enriched with annotations.
Side-conditions on these annotations in typing rules guarantee deadlock-freedom for cyclically connected processes.
The general idea of the annotations and side-conditions is that they can only be satisfied when the process is free from circular dependencies.

Based on Wadler's interpretation of classical linear logic as \CP (Classical Processes) \cite{journal/jfp/Wadler14} and Kobayashi's approach to deadlock-freedom~\cite{conf/concur/Kobayashi06}, Dardha and Gay~\cite{conf/fossacs/DardhaG18} presented \PCP (Priority-based \CP).
Their work reconciles session type theory's Curry-Howard foundations in linear logic with Kobayashi's more expressive approach.

Dardha and Gay refer to the annotations on types as \emph{priorities}, as they are numbers that indicate a global ordering on communications.
The idea is that the types of communications that block other communications can only be annotated with lower priority than the priority annotations of the blocked communications.
Duality is then extended to dictate that priority annotations on complementary communications should coincide.
This way, the \emph{local} ordering of communications is transferred between parallel processes when connecting endpoints, imposing a \emph{global} ordering of communications.
As long as the ordering is satisfiable, the process is guaranteed to be free from circular dependencies, and thus deadlock-free.

Let us reconsider process~$P_6$~\eqref{eq:intro:P6}.
Recall:
\[
    \eqIntroPsix*
\]
Let us list the session types assigned to endpoints in the left and right subprocesses of~$P_6$, this time including priority annotations ($\tEnd$ does not entail a communication and so does not require an annotation):
\begin{align*}
    x &: {!}^\pri (\tEnd) . \tEnd
    &
    z &: {?}^\pi (\tEnd) . \tEnd
    &
    w &: {!}^\rho (\tEnd) . \tEnd
    &
    y &: {?}^\gamma (\tEnd) . \tEnd
\end{align*}
Here $\pri,\pi,\rho,\gamma$ are variables representing numbers.
The order in which the associated endpoints are used determines conditions on the relation between these numbers, that may or may not be satisfiable.
Since the send on $x$ blocks the receive on $z$, and the send on $w$ blocks the receive on $y$, we have the following conditions on these priorities:
\begin{align*}
    \pri &< \pi
    &
    \rho &< \gamma
\end{align*}
Finally, since $x$ and $y$ should be typed dually, we require $\pri = \gamma$.
Similarly, we need $\pi = \rho$.
In the end, we have the following condition on priorities:
\[
    \pri < \pi < \pri
\]
This condition is clearly unsatisfiably.
Hence, we have detected the circular dependency in~$P_6$ and consider it ill-typed.

To contrast, let us reconsider process~$P_7$~\eqref{eq:intro:P7}.
We deemed~$P_7$ ill-typed using Rule~\ruleLabel{cut}, yet it is deadlock-free.
Recall:
\[
    \eqIntroPseven*
\]
We assign the same session types and priorities to the endpoints in the subprocesses of~$P_7$ as we did for~$P_6$ above.
Yet, because the receive on $y$ is now blocking the send on $w$ (instead of the other way around), we no longer require that $\pi < \pri$, and hence only require that $\pri < \pi$.
Clearly, this condition is indeed satisfiable, and thus we deem $P_7$ free from circular dependencies and well-typed, i.e., deadlock-free.

\paragraph{The influence of synchronous versus asynchronous communication.}

Thus far, I have only considered \emph{synchronous} communication, under which sends block whatever process they prefix until their corresponding receives are ready and communication has taken place (and similarly for selections/branches).
Its counterpart, \emph{asynchronous} communication, is an important theme in practice and in this thesis.
Under asynchronous communication, sends and selections do not block the processes they prefix.

Consider the following example, where two subprocesses are prefixed by a send:
\[
    P_8 \deq \pOut* z[u] ; \pRes{xy} ( \pOut* x[z] ; \0 \| \pIn y(w) ; \0 )
\]
Process~$P_8$ is stuck under synchronous communication, because there is no corresponding receive for the send on $z$, so it blocks the communication between $x$ and $y$.
In contrast, under asynchronous communication, the send on $z$ is non-blocking, so the communication between $x$ and $y$ can take place:
\[
    P_8 \redd \pOut* z[u] ; \0
\]

Being the standard mode of communication in practice, asynchrony is an important aspect of message-passing concurrency.
This leads me to the following research question:

\requAPCP*

\noindent
In \Cref{c:APCP}, I describe my answer to this question.

More precisely, I adapt Dardha and Gay's \PCP~\cite{conf/fossacs/DardhaG18} to the asynchronous setting to define \APCP (Asynchronous \PCP).
Let us once again recall process~$P_6$~\eqref{eq:intro:P6}:
\[
    \eqIntroPsix*
\]
Under asynchronous communication, $P_6$ is not deadlocked: neither of the sends are blocking, so communications can take place.
To see how \APCP adapts \PCP's priority analysis to guarantee deadlock-freedom including, e.g., $P_6$, let us recall its session type assignments:
\begin{align*}
    x &: {!}^\pri (\tEnd) . \tEnd
    &
    z &: {?}^\pi (\tEnd) . \tEnd
    &
    w &: {!}^\pi (\tEnd) . \tEnd
    &
    y &: {?}^\pri (\tEnd) . \tEnd
\end{align*}
Since sends are non-blocking in \APCP, they do not influence the local ordering of communications.
Therefore, the priority conditions $\pri < \pi$ and $\pi < \pri$ that should hold in \PCP are unnecessary in \APCP.
Hence, all priority conditions (of which there are none) are satisfiable, so this time $P_6$ is consider free from circular dependencies and well-typed: $P_6$ is deadlock-free under asynchronous communication.

\APCP shows that asynchrony enables more deadlock-free communication patterns, compared to, e.g., \PCP's synchrony.
%
\APCP and its cyclically connected processes with asynchronous communication will also appear in \Cref{c:LASTn,c:APCP}, where they are used as a basis for the analysis of sequential programming and multiparty session types, respectively, both discussed later in this introduction.

\subsubsection{Non-determinism}
\label{s:intro:contrib:pi:clpi}

The aspects of message-passing concurrency discussed above in \Cref{s:intro:contrib:pi:APCP} are important and fundamental.
Another important and fundamental aspect of message-passing concurrency is \emph{non-determinism}.
The form of branching in the \picalc discussed above is \emph{deterministic}, in that the exact choice that will be made will always be hard-coded in the processes.
Traditionally, the \picalc also features non-deterministic choices, where the outcome of the choice is left unspecified.
Non-deterministic choices encode real-world scenarios, where a program's environment influences the program's behavior.
This way, we can model, e.g., user interaction or internet connections that might fail.

Consider the following example, where one subprocess non-deterministically makes a selection on another subprocess' branch:
\begin{align}
    P_9 \deq \pRes{xy} \big( \pBra* x > {\{ \sff{left} : \0 , \sff{right} : \0 \}} \| ( \pSel* y < \sff{left} ; \0 + \pSel* y < \sff{right} ; \0 ) \big)
    \label{eq:intro:P9}
\end{align}
The left subprocess of~$P_9$ offers on $x$ a choice between branches labeled \sff{left} and \sff{right}.
Unlike before, the right subprocess does not hard-code a selection on $y$; instead, it offers a non-deterministic choice, denoted `$\,+\,$', between selecting \sff{left} and \sff{right}.
As a result, the behavior of $P_9$ is to either communicate the label \sff{left} or \sff{right}.

This form of non-deterministic choice is difficult to reconcile with the linear logic foundations of the session-typed \picalc.
This is because linear logic is strict in the linear treatment of resources.
In contrast, non-determinism might lead to the discarding of resources, similar to how unsafe deterministic branches might discard resources.
For example, let us adapt process~$P_5$~\eqref{eq:intro:P5} with non-determinism:
\[
    P'_5 \deq \pRes{xy} \big( \pIn x(v) ; \0 \| ( \pOut* y[u] ; \0 + \0 ) \big)
\]
The right subprocess of $P'_5$ might or might not send on $y$; the outcome is non-deterministic.
However, the left subprocess of $P'_5$ always intends to receive on $x$.
Hence, the non-de\-ter\-mi\-nis\-tic choice in $P'_5$ might lead to unexpected, unsafe behavior.

Another issue in reconciling non-determinism with logical foundations is the lack of \emph{confluence}.
Confluence means that different paths of computation always end in the same result.
For example, resolving addition is confluent: we have $3+1+4 = 4+4 = 8$ but also $3+1+4 = 3+5 = 8$.
On the contrary, non-deterministic choice in the \picalc is not confluent.
For example, consider an adaptation of process~$P_9$~\eqref{eq:intro:P9} where the branching subprocess makes a different selection on a new endpoint $z$ in each branch:
\begin{align}
    P_{10} \deq \pRes{xy} \big( \pBra* x > {\{ \sff{left} : \pSel* z < \sff{left} ; \0 , \sff{right} : \pSel* z < \sff{right} ; \0 \}} \| ( \pSel* y < \sff{left} ; \0 + \pSel* y < \sff{right} ; \0 ) \big)
    \label{eq:intro:P10}
\end{align}
In the left subprocess of $P_{10}$, the label received on $x$ determines the label sent on $z$.
But this all depends on the non-deterministic choice in the right subprocess on which label is sent on~$y$.
Hence, the non-deterministic choice in the right subprocess leads to different outcomes, breaking confluence:
\begin{align}
    P_{10} &\redd \pSel* z < \sff{left} ; \0
    & &\text{or}
    &
    P_{10} &\redd \pSel* z < \sff{right} ; \0
    \label{eq:intro:P10Redd}
\end{align}

To reconcile non-determinism with confluence and Curry-Howard interpretations of linear logic as session type systems for the \picalc, Caires and Pérez~\cite{conf/esop/CairesP17} introduced a new non-deterministic (internal) choice operator that does not \emph{commit} to a choice: $P \oplus* Q$.
This way, if $P \redd P'$ and $Q \redd Q'$, then $P \oplus* Q \redd P' \oplus* Q'$: both branches behave independently.
Let us further adapt process~$P_{10}$~\eqref{eq:intro:P10} to illustrate the effect of this new operator:
\begin{align*}
    P_{11} &\deq \pRes{xy} \big( \pBra* x > {\{ \sff{left} : \pSel* z < \sff{left} ; \0 , \sff{right} : \pSel* z < \sff{right} ; \0 \}} \| ( \pSel* y < \sff{left} ; \0 \oplus* \pSel* y < \sff{right} ; \0 ) \big)
    \\
    &\redd \pSel* z < \sff{left} \oplus* \pSel* z < \sff{right}
\end{align*}
Thus, the non-deterministic choice in the right subprocess of $P_{11}$ does not commit to a choice; instead, after the communication of the label between $x$ and $y$, which label to send on $z$ remains a non-deterministic choice.

Although Caires and Pérez's non-deterministic choice operator maintains strong ties with linear logic foundations through confluence, it arguably is not an adequate representation of real-world non-determinism.
In practice, e.g., communication channels may randomly and unrecoverably fail, or a user may select an operation that is supposed to have a different outcome than another operation.
That is, realistic non-determinism favors \emph{commitment} over confluence.
This leads me to the following research question:

\requND*

\noindent
In \Cref{c:clpi}, I propose an approach to answer this question.

To be precise, my answer is the design of a session-typed \picalc \clpi with a new operator for non-deterministic choice, denoted `$\,\nd\,$'.
The calculus \clpi is equipped with two different semantics that explore the tension between confluence and commitment.

Opposed to Caires and Pérez's $\,\oplus*\,$ which sits at the extreme of confluence, \clpi has a so-called \emph{eager} semantics that sits at the extreme of commitment.
This way, the new operator $\,\nd\,$ behaves similarly to the traditional non-deterministic choice $\,+\,$.
That is, given
\begin{align}
    P_{12} \deq \pRes{xy} \big( \pBra* x > {\{ \sff{left} : \pSel* z < \sff{left} ; \0 , \sff{right} : \pSel* z < \sff{right} ; \0 \}} \| ( \pSel* y < \sff{left} ; \0 \nd \pSel* y < \sff{right} ; \0 ) \big),
    \label{eq:intro:P12}
\end{align}
compare the following to~\eqref{eq:intro:P10Redd}:
\begin{align}
    P_{12} &\redd \pSel* z < \sff{left} ; \0
    & &\text{or}
    &
    P_{12} &\redd \pSel* z < \sff{right} ; \0
    \label{eq:intro:P12Redd}
\end{align}

Perhaps more interesting is \clpi's so-called \emph{lazy} semantics that sits right between confluence and commitment.
To be more precise, the lazy semantics postpones commitment until truly necessary, i.e., until a choice has to be made.
For example, consider the following extension of $P_{12}$~\eqref{eq:intro:P12} where the selections on $y$ are now prefixed by a send:
\[
    P_{13} \deq \pRes{xy} \begin{array}[t]{@{}l@{}}
        \big( \pIn x(u) ; \pBra* x > {\{ \sff{left} : \pSel* z < \sff{left} ; \0 , \sff{right} : \pSel* z < \sff{right} \}}
        \\
        {} \| ( \pOut* y[w] ; \pSel* y < \sff{left} ; \0 \nd \pOut* y[w] ; \pSel* y < \sff{right} ) \big)
    \end{array}
\]
Notice that both branches of the non-deterministic choice in the right subprocess of~$P_{13}$ start with a send on $y$.
Hence, the lazy semantics postpones choosing a branch, resulting in $P_{13} \redd P_{12}$.
Since in $P_{12}$ there is an actual choice, the lazy semantics commits as in~\eqref{eq:intro:P12Redd}.
Compare this to how the eager semantics immediately picks a branch in $P_{14}$:
\begin{align*}
    P_{14} &\redd \pRes{xy} ( \pBra* x > {\{ \sff{left} : \pSel* z < \sff{left} ; \0 , \sff{right} : \pSel* z < \sff{right} ; \0 \}} \| \pSel* y < \sff{left} )
    \intertext{or}
    P_{14} &\redd \pRes{xy} ( \pBra* x > {\{ \sff{left} : \pSel* z < \sff{left} ; \0 , \sff{right} : \pSel* z < \sff{right} ; \0 \}} \| \pSel* y < \sff{right} )
\end{align*}
Hence, under the eager semantics the initial communication on $x$ and $y$ determines later choices, whereas under the lazy semantics these choices are left open.

Under both eager and lazy semantics, \clpi guarantees protocol fidelity, communication safety, and deadlock-freedom by typing.
\Cref{c:clpi} also briefly discusses a faithful translation of a \lamcalc with non-determinism, as well as interesting formal comparisons between the eager and lazy semantics.

\subsubsection{Alternative Logical Foundations: Bunched Implications}
\label{s:intro:contrib:pi:piBI}

The logical foundations of session types in linear logic due to Caires and Pfenning~\cite{conf/concur/CairesP10} are an important theme in this dissertation.
Its fine-grained control of linear resources (that must be used exactly once) fits very well with the ideas behind session types for the \picalc.
Many follow-up works have established the robustness of this Curry-Howard isomorphism (see, e.g., \cite{journal/jfp/Wadler14,conf/csl/DeYoungCPT12,journal/ic/PerezCPT14,conf/icfp/BalzerP17}).

However, linear logic is not the only logic that is suitable for reasoning about resources for concurrent programming.
In particular, O'Hearn and Pym's logic of bunched implications (\BI)~\cite{journal/bsl/OHearnP99} is another substructural logic with fine-grained resource control.
Instead of a focus on \emph{number-of-uses} as in linearity, \BI focuses on \emph{ownership} and the \emph{provenance} (i.e., origin) of resources.
Nonetheless, the inference systems of \BI and linear logic are rather similar.
Hence, this leads me to the next research question:

\requBI*

\Cref{c:piBI} gives an answer to this question by presenting \pibi, a variant of the \picalc with a Curry-Howard interpretation of \BI as session type system.
To give an overview of the design of \pibi, let me give some more context about the logic \BI.

Linear logic has one connective for each way of building propositions: one conjunction, one implication, and so forth.
In contrast, \BI has two connectives for each way of building propositions, that can be freely combined: additive and multiplicative.
For example, \BI has two forms of conjunction: `$\,\wedge\,$' (additive) and `$\,\sep\,$' (multiplicative).

Linear logic propositions all follow the same structural principles: linearity dictates that resources may not be contracted (duplicated) or weakened (discarded).
The case for \BI is more complicated: additive and multiplicative connectives follow different structural principles, i.e., additively combined resources may be contracted and weakened, but multiplicatively combined resources may not.
This leads to the fine-grained control of ownership and provenance of resources, though at this point it will probably not be clear how this translates to session types and message-passing.

The connectives (additive and multiplicative) of \BI are very similar to those of linear logic, and so are their associated inference rules.
It is thus unsurprising that the interpretation of \BI's inference rules for connectives as session typing rules is very similar to the interpretation of linear logic's inference rules.
The true challenge in designing \pibi was how to interpret the structural rules for the contraction and weakening of additively combined resources.

The result is a new construct called ``spawn''.
Spawn is related to programming principles for resource aliasing and pointers.
It allows processes to safely duplicate and discard sessions.

As an example of session duplication, consider the following process that receives on two endpoints: $\pIn y_1(v) ; \pIn y_2(w) ; \0$
Now suppose we only have a process available that sends on a single endpoint: $\pOut* x[z] ; \0$.
We can use the spawn construct to connect both endpoints of the former process to the single endpoint of the latter by contracting the two endpoints into one:
\[
    \pSpw [y->y_1,y_2] ; \pIn y_1(v) ; \pIn y_2(w) ; \0.
\]
This process is prefixed with the spawn construct `$\,\pSpw [\ldots]\,$', which in this case says that $y$ will be used twice, as $y_1$ and as $y_2$.
Now we can connect the single endpoint $y$ to $x$; the result is that the spawn construct requests two copies of the send on $x$:
\begin{align*}
    & \pRes{xy} ( \pOut* x[z] ; \0 \| \pSpw [y->y_1,y_2] ; \pIn y_1(v) ; \pIn y_2(w) ; \0 )
    \\
    {} \redd {} & \pRes{x_1y_1} ( \pOut* x_1[z_1] ; \0 \| \pRes{x_2y_2} ( \pOut* x_2[z_2] ; \0 \| \pIn y_1(v) ; \pIn y_2(w) ; \0 ) )
\end{align*}
Thus, the send on $x$ is copied twice, renamed and connected appropriately.

To illustrate the discarding of sessions, let us recall process~$P_5$~\eqref{eq:intro:P5}:
\[
    \eqIntroPfive*
\]
This process is considered unsafe, because the send on $y$, expected by the receive on $x$, may not be available.
The \sff{quit}-branch in the right subprocess of~$P_5$ can be seen as the discarding of the session on $x$ and $y$.
Hence, we can use the spawn construct to make a safe version of~$P_5$, this time pointing~$y$ to nothing to weaken it, denoted $\pSpw [y->\emptyset] ; \0$.
As a result, if the \sff{quit}-branch is taken, the receive on $x$ is discarded:
\begin{align*}
    & \pRes{zw} ( \pSel* z < \sff{quit} ; \0 \| \pRes{xy} ( \pIn x(v) ; \0 \| \pBra* w > {\{ \sff{send} ; \pOut* y[u] ; \0 , \sff{quit} : \pSpw [y->\emptyset] ; \0 \}} ) )
    \\
    {} \redd {} & \pRes{xy} ( \pIn x(v) ; \0 \| \pSpw [y->\emptyset] ; \0 )
    \\
    {} \redd {} & \0
\end{align*}
Thus, adding the spawn construct to~$P_5$ makes the process safe.

Duplicating and discarding processes requires care: if we duplicate/discard a session provided by a process that relies on further sessions, those sessions become unsafe.
Consider the following example, where the right subprocess requests two copies of the left subprocess, which relies on further sessions:
\begin{align*}
    & \pRes{xy} ( \pSel* u < \ell ; \pOut* x[z] ; \0 \| \pSpw [y->y_1,y_2] ; \pIn y_1(v) ; \pIn y_2(w) ; \0 )
    \\
    {} \redd {} & \pRes{x_1y_1} ( \pSel* u < \ell ; \pOut* x_1[z_1] ; \0 \| \pRes{x_2y_2} ( \pSel* u < \ell ; \pOut* x_1[z_2] ; \0 \| \pIn y_1(v) ; \pIn y_2(w) ; \0 ) )
\end{align*}
Where initially there was one selection on $u$, there are now two.
The initial process would need to connect $u$ to an endpoint doing a single branch, so after the duplication this session becomes unsafe.
The solution is so-called \emph{spawn propagation}: we create a new spawn construct that further duplicates/discards sessions on which the duplicated/discarded process relies.
This way, we make the example above safe by propagating the spawn:
\begin{align*}
    & \pRes{xy} ( \pSel* u < \ell ; \pOut* x[z] ; \0 \| \pSpw [y->y_1,y_2] ; \pIn y_1(v) ; \pIn y_2(w) ; \0 )
    \\
    {} \redd {} & \pSpw [u->u_1,u_2] ; \pRes{x_1y_1} ( \pSel* u_1 < \ell ; \pOut* x_1[z_1] ; \0 \| \pRes{x_2y_2} ( \pSel* u_2 < \ell ; \pOut* x_1[z_2] ; \0 \| \pIn y_1(v) ; \pIn y_2(w) ; \0 ) )
\end{align*}
Hence, if we initially connect $u$ to an endpoint doing a single branch, the propagated spawn will duplicate the connected process, making the session safe.

The calculus \pibi enjoys the correctness properties we desire: session fidelity, communication safety, and deadlock-freedom.
In the case of \pibi, deadlock-freedom is guaranteed because it uses the \BI Rule~\ruleLabel{cut} (cf.\ the discussion of the linear logic Rule~\ruleLabel{cut} in the introduction of \Cref{s:intro:contrib:pi}).
Additionally, \pibi is weakly normalizing, meaning that any well-typed process with all endpoints connected can always complete all its communications until it is done.

\subsection{Binary Session Types and Functional Programming (\crtCref{p:lambda})}
\label{s:intro:contrib:lambda}

The abstraction of message-passing programs introduced in \Cref{s:intro:contrib:pi} might feel somewhat foreign for those used to real-world programming languages.
In this section I introduce \Cref{p:lambda}, which discusses message-passing aspects of ``functional'' or ``sequential'' programming.
In particular, I will discuss variants of the well-known \lamcalc.

The \lamcalc is a simple calculus of \emph{functions}.
For example, the term $(\lam x . x)~y$ denotes a function that takes a parameter $x$, (denoted $\lam x \ldots$) and returns it, and applies it to the variable $y$.
As a result, we substitute the functions' parameter $x$ for $y$ and return the function's body: 
\[
    (\lam x . x)~y \redd y.
\]

As is the case for the \picalc, the \lamcalc can be extended with constructs to model all sorts of programming features.
Traditionally, extensions of the \lamcalc have been used as prototype programming languages to illustrate the features of newly introduced programming formalisms.
A (faithful) translation from the extended \lamcalc into the new formalism then shows that its features are suitable for integration in real programming.

The idea of a \lamcalc as a prototype programming language has always been a tradition in the line of work based on the \picalc, already starting at the conception of the \picalc by Milner \etal~\cite{journal/ic/MilnerPW92}.
This way, a (faithful) translation from a \lamcalc to a \picalc serves as the ultimate \emph{litmus test} of the capabilities of the \picalc.

\subsubsection{Asynchronous Message-passing and Cyclically Connected Threads}
\label{s:intro:contrib:lambda:LASTn}

The \lamcalc does not have to be limited to sequential computation.
In particular, there are many works that incorporate (session-typed) message-passing in functional programming formalisms, a line of work starting with~\cite{report/GayVR03,conf/concur/VasconcelosRG04,conf/padl/NeubauerT04,journal/tcs/VasconcelosGR06,report/GayV07}.
The first account is by Gay and Vasconcelos~\cite{journal/jfp/GayV10}.
They introduced a \lamcalc, referred to in this thesis as \LAST, with threads that are executed concurrently and communicate through message-passing over channels.

Message-passing in \LAST is asynchronous: sends and selects place messages in buffers without blocking their continuations, and receives and branches read these messages from the buffers.
Moreover, message-passing in \LAST is typed with linear session types.
As a result, well-typed \LAST terms enjoy session fidelity and communication safety, but not deadlock-freedom as threads may be connected cyclically and there is no mechanism to detect and reject circular dependencies (cf.\ \Cref{s:intro:contrib:pi}).

In his account of classical linear logic as a Curry-Howard foundation for the session-typed \picalc \CP, Wadler introduced a variant of \LAST called \GV (for Good Variation)~\cite{conf/icfp/Wadler12,journal/jfp/Wadler14}.
\GV features synchronous communication, and, based on linear logic's Rule~\ruleLabel{cut}, does not permit cyclically connected threads.
This way, \GV served as a litmus test for Wadler's \CP through a typed translation from the former to the latter.

Many works followed suit, accompanying new session-typed variants of the \picalc with variants of \LAST and providing a translation.
In particular, Kokke and Dardha designed \PGV (Priority \GV)~\cite{conf/forte/KokkeD21}, based on Dardha and Gay's \PCP~\cite{conf/fossacs/DardhaG18} (discussed in \Cref{s:intro:contrib:pi:APCP}).
\PGV extends Wadler's \GV with cyclically connected threads, and adds priority annotations and conditions to its session type system.
This way, well-typed \PGV terms are free from circular dependencies and thus deadlock-free.

Given the step from \PCP to \APCP (from synchronous to asynchronous message-passing) described in \Cref{s:intro:contrib:pi:APCP}, it is only natural to wonder whether something similar can be done from \PGV.
This leads me to my next research question:

\requLASTn*

The answer to this question is detailed in \Cref{c:LASTn}.
However, the answer does not entirely follow the step from \PCP to \APCP as you might expect, i.e., a calculus A\PGV (Asynchronous \PGV).
Instead, the result is the calculus \LASTn (call-by-name \LAST).

Besides not having priority annotations in its session type system (something I will come back to soon), \LASTn adopts a different semantics than \PGV and its predecessors.
To make this more precise, let me digress into the usual strategies for function application in the \lamcalc: call-by-value and call-by-name.
Under call-by-value semantics ($\reddV$), functions may only be applied once its parameters are fully evaluated; for example:
\[
    (\lam x . x)~\big((\lam y . y)~z\big) \reddV (\lam x . x)~z \reddV z
\]
That is, the right function application needs to be evaluated before applying the left function.
In contrast, under call-by-name semantics ($\reddN$), a function's parameters may not be evaluated until they have been substituted into the function's body; for example:
\[
    (\lam x . x)~\big((\lam y . y)~z\big) \reddN (\lam y . y)~z \reddN z
\]
That is, the left function needs to be applied before evaluating the right function application.

Where \PGV and its predecessors adopt call-by-value semantics, \LASTn adopts call-by-name semantics; this explains the name ``\LASTn'': it is a call-by-name variant of \LAST.
The reason for this important design decision is rooted in a desire for a strong connection with \APCP as stated in \Cref{requ:LASTn}; I will get back to this after discussing message-passing in \LASTn.

Inspired by \LAST, \LASTn features threads that communicate on channel endpoints connected by message buffers.
In the following, two threads exchange a term on a channel:
\[
    \pRes{x\tBfr{\epsi}y} ( \tSend (M,x) \prl \tRecv y ) \redd \pRes{x\tBfr{M}y} ( x \prl \tRecv y ) \redd \pRes{x\tBfr{\epsi}y} ( x \prl (M,y) )
\]
Here are two parallel threads (separated by `$\,\prl\,$').
The endpoints $x$ and $y$ are connected by a buffer (denoted $\pRes{x\tBfr{\ldots}y}$) that is empty (denoted $\epsilon$).
The left thread sends on $x$ some term~$M$, which ends up in the buffer.
In a following step, the right thread receives $M$ on $y$.

As I mentioned before, the session type system of \LASTn does not include priorities (as in \PGV).
Hence, \LASTn ensures session fidelity and communication safety by typing, but permits circular dependencies and so does not guarantee deadlock-freedom.
This was a conscious design decision: a faithful translation from \LASTn to \APCP allows us to leverage \APCP's priorities to recover deadlock-freedom for \LASTn indirectly, while keeping \LASTn and its type system relatively uncomplicated.

The design of \LASTn and its call-by-name semantics are directly inspired by a translation from \LASTn to \APCP.
The translation is faithful in that it \emph{preserves typing} and is \emph{operationally correct}.
The former means that the types of a source term are themselves translated to \APCP types, albeit without priority annotations.
The latter is an important property of translations coined by Gorla~\cite{journal/ic/Gorla10} that signifies that the behavior of translated processes concurs precisely with the behavior of their source terms.

Thus, the translation from \LASTn to \APCP preserves typing modulo priority annotations.
Yet, it is possible to annotate the types of translated \APCP processes with priorities.
If the associated conditions on priorities are satisfiable, then the translated \APCP processes are deadlock-free.
The operational correctness of the translation then allows us to infer that the source \LASTn terms of these processes are also deadlock-free.
Hence, deadlock-freedom does hold for the subset of \LASTn terms that translate to \APCP processes with satisfiable priority annotations.
There is strong evidence that such an approach to deadlock-freedom would not be possible if \LASTn had call-by-value semantics; hence, the call-by-name semantics.

\subsubsection{Bunched Functions: a Litmus Test for \texorpdfstring{\pibi}{piBI}}
\label{s:intro:contrib:lambda:alphalambda}

The message-passing calculus \pibi (described in \Cref{s:intro:contrib:pi:piBI}) is not the first Curry-Howard interpretation of \BI.
Hinted at in O'Hearn and Pym's first introduction of \BI~\cite{journal/bsl/OHearnP99}, O'Hearn introduces the \alamcalc~\cite{journal/jfp/OHearn03}.
The \alamcalc is a variant of the \lamcalc derived from \BI's natural deduction system, where propositions are functional types, proofs are terms, and proof normalization is $\beta$-reduction (i.e., function application and variable substitution).

The \alamcalc is a simple calculus with only function abstraction and application.
However, what sets the \alamcalc apart is that it has two kinds of functions.
Consider:
\begin{align}
    \lam x . \alpha f . (f \at x) \at x
    \label{eq:intro:alam:unusual}
\end{align}
Here, $\lam x \ldots$ denotes a ``linear'' function and $\alpha f \ldots$ an ``unrestricted'' one; the $\at$-symbol denotes application of an unrestricted function.
O'Hearn~\cite{journal/jfp/OHearn03} calls this example ``unusual''.
Indeed, especially to those familiar with variants of the \lamcalc with linearity, it is: the argument of a linear function `$\,x\,$' is used twice.

However, the example above only seems unusual under the interpretation of linearity that says that resources must be used exactly once.
Hence, the quotation marks around ``linear'' above: in the \alamcalc, the restriction of linear functions is not on how many times their arguments may be used, but on that they may not be applied to duplicated arguments.
To contrast, consider the following example:
\[
    \alpha f . \lam x . (x~f)~f
\]
Here, the ``linear'' and the ``unrestricted'' functions have swapped places compared to the example above~\eqref{eq:intro:alam:unusual}.
Under the usual interpretation of linearity, this example is fine: the argument `$\,x\,$' of the linear function is used exactly once.
However, it cannot be typed in the \alamcalc, because the linear function is applied to a duplicated resource `$\,f\,$'.

These examples show that the \alamcalc is a useful prototyping language for illustrating \BI's approach to fine-grained resource management.
Hence, the \alamcalc would be an excellent candidate for a translation to \pibi.
As discussed in the introduction to \Cref{s:intro:contrib:lambda}, such a translation would serve as a litmus test for the design and capabilities of \pibi, especially because the \alamcalc has not been designed with translations to \pibi in mind.
This leads me to the next research question:

\requAL*

\Cref{c:alphalambda} answers this question positively, by means of a typed translation from the \alamcalc to \pibi.
The translation is rather straightforward: it follows canonical translations of proofs in natural deduction into sequent calculi (cf., e.g., \cite[Section~6.3]{book/Pym13}), very similar to the well-known translations from variants of the \lamcalc to variants of the \picalc by Milner~\cite{journal/mscs/Milner92}, Sangiorgi and Walker~\cite{book/SangiorgiW03}, and Wadler~\cite{journal/jfp/Wadler14}.

The challenging part of the translation is handling the contraction (duplication) and weakening (discarding) of variables.
Without it, the translation would not really be a litmus test for \pibi.
Indeed, in \pibi, contraction and weakening are interpreted with the spawn-construct.
So, the real litmus test here is whether contraction and weakening in the \alamcalc translates faithfully to applications of spawn in \pibi.
Fortunately, the result is uncomplicatedly positive.

The translation presented in \Cref{c:alphalambda} is faithful.
This means that is satisfies two important properties.
Firstly, the translation preserves typing: translated \pibi processes are typed identically to their respective source \alamcalc terms (compare this to the translation introduced in \Cref{s:intro:contrib:lambda:LASTn}, where \LASTn types are translated to suitable \APCP types).
Secondly, the translation is operationally correct, i.e., the behavior of translated processes concurs precisely with that of their respective source terms.

\subsection{Multiparty Session Types (\crtCref{p:mpst})}
\label{s:intro:contrib:mpst}

Thus far, I have talked about \emph{binary} session types, that describe communication protocols between pairs of participants.
In practice, distributed systems often comprise multiple components, whose interactions are more intertwined than faithfully representable by binary protocols.
Hence, \Cref{p:mpst} of my dissertation concerns \emph{multiparty session types} (\MPSTs).

Introduced by Honda \etal~\cite{conf/popl/HondaYC08,journal/acm/HondaYC16}, \MPST is a theory for communication protocols with two \emph{or more} participants.
Such protocols are usually expressed as \emph{global types}, describing the interactions between participants \emph{from a vantage point}.

Let us consider a variant of the well-known example called ``The Two-buyer Protocol''; it describes how Alice ($a$) and Bob ($b$) together buy a book from Seller ($s$).
First, $a$ sends to $s$ the book's title.
Then, $s$ sends to $a$ and $b$ the book's price.
Finally, $b$ sends to $s$ a choice between buying the book or quiting, after which the protocol ends.
We can formalize this protocol as the following global type:
\begin{align}
    G_{\sff{2b}} \deq a {!} s (\msg \sff{title}<\sff{str}>) . s {!} a (\msg \sff{price}<\sff{int}>) . s {!} b (\msg \sff{price}<\sff{int}>) . b {!} s \{ \sff{buy} . \tEnd , \sff{quit} . \tEnd \}
    \label{eq:intro:G2b}
\end{align}
Here, $a {!} s (\ldots)$ denotes that $a$ sends to $s$ a messages.
Messages are of the form $\msg \sff{label}<\sff{type}>$: they carry a label and a value of the given type.
Another form of message is choice, denoted $b {!} s \{ \ldots \}$, where $b$ sends to $s$ either of the given labels, and the protocol proceeds accordingly.
The type $\tEnd$ denotes the end of the protocol.

As originally intended by Honda \etal, \MPSTs are designed to verify protocol conformance in practical settings, where the roles of protocol participants are implemented \emph{distributedly} across networks of components that communicate \emph{asynchronously}.
Here, distributed means that the theory makes no assumptions about how exactly components are connected: whether they communication directly in pairs or through a centralized unit of control does not matter.
Moreover, asynchrony means that it should be possible to establish protocol conformance without requiring components to wait for their messages to be received.

\subsubsection{A Relative Perspective of Multiparty Session Types}
\label{s:intro:contrib:mpst:relProj}

A significant benefit of the distributed and asynchronous nature of \MPSTs is that correctness verification is \emph{compositional}:  the correctness of the whole network of components can be guaranteed by verifying the correctness of each component separately.
Let me refer to a component implementing/modeling the role of a specific protocol participant's role as simply a ``participant''.

To verify the behavior of a participant we need to focus on their precise contributions in the overal protocol.
Using global types directly make this an overly complicated task, as they contain ``noise'' from other participants.
It therefore makes sense to instead work with a protocol that provides a perspective of the overal protocol that is \emph{local} to the participant.

Honda \etal~\cite{conf/popl/HondaYC08,journal/acm/HondaYC16} obtain such a local perspective by a \emph{local projection} of a global type onto a participant, resulting in a \emph{local type}.
For example, the projection of~$G_{\sff{2b}}$~\eqref{eq:intro:G2b} onto participant $a$ is as follows:
\[
    a {!} s (\msg \sff{title}<\sff{str}>) . a {?} s (\msg \sff{price}<\sff{int}>) . \tEnd
\]
Here, `$\,{?}\,$' denotes receiving.
The role of $a$ in $G_{\sff{2b}}$ is thus to send to $s$ a book's title, then to receive from $s$ the book's price, and then their role in the protocol ends.

Local projection gets more complicated as protocols get more interesting.
Consider the following adaptation of $G_{\sff{2b}}$~\eqref{eq:intro:G2b}, where $b$ informs $a$ of their choice after informing~$s$:
\begin{align}
    G'_{\sff{2b}} \deq a {!} s (\msg \sff{title}<\sff{str}>) . s {!} a (\msg \sff{price}<\sff{int}>) . s {!} b (\msg \sff{price}<\sff{int}>) . b {!} s \left\{
        \begin{array}{@{}l@{}}
            \sff{buy} . b {!} a (\sff{buy}) . \tEnd ,
            \\
            \sff{quit} . b {!} a (\sff{quit}) . \tEnd
        \end{array}
    \right\}
    \label{eq:intro:G2bp}
\end{align}
The local projection of $G'_{\sff{2b}}$ onto $a$ requires care, because the label they are supposed to receive from $b$ at the end of the protocol depends on the message sent by $b$ to $s$ before.
That is, $a$'s protocol depends on a choice outside of their protocol, referred to as a \emph{non-local} choice.

Honda \etal's local projection cannot deal with global types such as $G'_{\sff{2b}}$, because it requires the projections of the branches of non-local choices to be indistinguishable; this de facto rules out any non-local choices.
To support more interesting protocols, Yoshida \etal~\cite{conf/fossacs/YoshidaDBH10} introduced a variant of local projection with an operator that \emph{merges} the branches of non-local choices as long as they all are receives from the same participants.
Using merge, the local projection of $G'_{\sff{2b}}$ onto $a$ is as follows:
\[
    a {!} s (\msg \sff{title}<\sff{str}>) . a {?} s (\msg \sff{price}<\sff{int}>) . a {?} b \{ \sff{buy} . \tEnd , \sff{quit} . \tEnd \}
\]
This local type is quite sensible: $a$ can be ready to receive \sff{buy} or \sff{quit}, independent of which label $b$ sent to $s$ before.

Let us now consider a perhaps even more interesting adaptation of $G_{\sff{2b}}$~\eqref{eq:intro:G2b}, where $a$ sends their address to $s$ when $b$ chooses to buy:
\begin{align}
    G''_{\sff{2b}} \deq a {!} s (\msg \sff{title}<\sff{str}>) . s {!} a (\msg \sff{price}<\sff{int}>) . s {!} b (\msg \sff{price}<\sff{int}>) . b {!} s \left\{
        \begin{array}{@{}l@{}}
            \sff{buy} . a {!} s (\msg \sff{addr}<\sff{str}>) . \tEnd ,
            \\
            \sff{quit} . \tEnd
        \end{array}
    \right\}
    \label{eq:intro:G2bpp}
\end{align}
This time, merge cannot resolve $b$'s non-local choice in $a$'s protocol, since it cannot reconcile a send and the end of the protocol.
This is unfortunate, because $G''_{\sff{2b}}$ represents a protocol that is very useful in practice.

It is apparent that global types can express protocols that cannot be handled by \MPST theories based on local projection (with or without merge).
Hence, such theories rely on \emph{well-formed} global types, defined by local projectability onto all participants.
A large portion of the literature on \MPSTs accepts the class of well-formed global types induced by local projection with merge as the standard, even though this class does not include practical protocol such as $G''_{\sff{2b}}$~\eqref{eq:intro:G2bpp}.

Scalas and Yoshida argue to reprimand this lack of expressivity by getting rid of global types and only working with local types (and their compatibility) directly~\cite{conf/popl/ScalasY19}.
Their solution allows for the verification of a much wider range of multiparty interactions, but their severance with global types may not be of interest for distributed implementations of multiparty protocols in practice.
This leads me to my next research question:

\requRelProj*

\ExecuteMetaData[c.relProj.tex]{intro:relProj:diagram}

\Cref{c:relProj} answers this question by introducing a new local perspective of global types: instead of individual participants' protocols, it considers protocols between \emph{pairs} of participants.
Such protocols are represented by so-called \emph{relative types}, which are obtained from global types through \emph{relative projection}.
\Cref{f:intro:relProj:diag} compares how local and relative projection decompose a global type with three participants, and how the resulting types are used to verify the behavior of three processes implementing the roles of the global type's participants.

A major advantage of relative types is their direct compatibility with binary session types, in contrast to local types which need another level of projection to be compatible.
However, the distinctive feature of relative types is their explicit treatment of non-local choices as \emph{dependencies}.
That is, if a global type contains a non-local choice, relative projection will encode the non-local choice as a dependency message in the relative type.

To illustrate relative types and their explicit dependency messages, let us consider the relative projection of $G''_{\sff{2b}}$~\eqref{eq:intro:G2bpp} onto $a$ and $s$ (i.e., the protocol representing the interactions in $G''_{\sff{2b}}$ between participants $a$ and $s$):
\[
    a {!} s (\msg \sff{title}<\sff{str}>) . s {!} a (\msg \sff{title}<\sff{str}>) . (s {?} b) {!} a \left\{
        \begin{array}{@{}l@{}}
            \sff{buy} . a {!} s (\msg \sff{addr}<\sff{str}>) . \tEnd ,
            \\
            \sff{quit} . \tEnd
        \end{array}
    \right\}
\]
Here, `$\,(s {?} b) {!} a\,$' denotes that $s$ receives from $b$ a choice between labels \sff{buy} and \sff{quit}, and that $s$ needs to forward the received choice to $a$.
This way, $a$ becomes aware of the branch of the non-local choice by $b$ in which they are, allowing them to proceed correctly according to the global protocol.

\Cref{c:relProj} then defines a new class of well-formed global types based on relative projection.
This class includes $G_{\sff{2b}}$~\eqref{eq:intro:G2b}, $G'_{\sff{2b}}$~\eqref{eq:intro:G2bp}, and $G''_{\sff{2b}}$~\eqref{eq:intro:G2bpp}.
However, there are global types that are well-formed under local projection that are not well-formed under relative projection, though \Cref{c:relProj} discusses how such protocols can be fixed for relative projection to support them.
The final two contributions of my dissertation, introduced next, rely heavily on the theory of relative projection defined in \Cref{c:relProj}.

\subsubsection{Analyzing Distributed Implementations of Multiparty Session Types}
\label{s:intro:contrib:mpst:mpstAPCP}

Studying how \MPSTs can be used in the verification of model implementations of protocols is important for bringing \MPSTs closer to practical application.
The literature on this subject is vast, with approaches from many angles.
The most salient approach is based on variants of the \picalc with type systems that use global and local types directly, such as in Honda \etal's initial development of \MPSTs~\cite{conf/popl/HondaYC08,journal/acm/HondaYC16}.

Type systems tailored for \MPSTs usually guarantee deadlock-freedom by typing, as global types do not induce communication patterns that can deadlock.
This is a major advantage of this approach, but there are downsides.
First and foremost, proving correctness guarantees for such systems is a complex task, as it requires reasoning on global and local levels at the same time.
As a result, extending existing systems with practical features such as error handling is quite an endeavor.
Second, \MPSTs are intended for distributed systems, but such type systems often abstract away from network configuration.
This makes it unclear how lower level aspects of network configuration may affect correctness guarantees.

Another approach is to connect \MPSTs with binary session types.
If done correctly, this enables the use of a vast range of binary session type systems with a multitude of features and correctness guarantees for the analysis of model implementations of multiparty protocols.
The idea, coined by Caires and Pérez~\cite{conf/forte/CairesP16} and by Carbone \etal~\cite{conf/concur/CarboneMSY15,conf/concur/CarboneLMSW16,journal/ai/CarboneMSY17}, is to translate local types to binary session types.
This makes it possible to use binary session type systems to guarantee correctness properties of models of protocol participants.
Moreover, meta-theoretical results guarantee that the entire system conforms to the global protocol.

The approaches by Caires and Pérez~\cite{conf/forte/CairesP16} and Carbone \etal~\cite{conf/concur/CarboneMSY15,conf/concur/CarboneLMSW16,journal/ai/CarboneMSY17} have a caveat that is pivotal to their practical applicability.
Both works rely on binary session type systems derived from linear logic, which do not allow processes to be cyclically connected (which indirectly guarantees deadlock-freedom, cf.\ \Cref{s:intro:contrib:pi:APCP}).
Then, the only way to compose processes modeling protocol participants is to connect them all to an additional process that forwards messages between the participants.
Such an additional process (called \emph{medium} in~\cite{conf/forte/CairesP16} and \emph{arbiter} in~\cite{conf/concur/CarboneMSY15,conf/concur/CarboneLMSW16,journal/ai/CarboneMSY17}) can be generated from a global type, ensuring that it distributes messages correctly.

Requiring an additional process to connect the participants and ``orchestrate'' their interaction defies the distributedness of multiparty systems, while it is essential for implementations of \MPSTs and real-world multiparty systems to be distributed systems.
Hence, it is important to come up with solutions that leverage binary session types while preserving distributedness.
Moreover, important correctness features such as deadlock-freedom and protocol conformance should be supported.
This leads me to the next research question:

\requMpstAPCP*

\Cref{c:mpstAPCP} argues that \APCP (introduced in \Cref{s:intro:contrib:pi:APCP}) is up to the task: it guarantees deadlock-freedom for cyclically connected, thus enabling the analysis of distributed model implementations of \MPSTs.
Using \APCP overcomes another shortcoming of~\cite{conf/forte/CairesP16,conf/concur/CarboneMSY15,conf/concur/CarboneLMSW16,journal/ai/CarboneMSY17}: their systems rely on synchronous communication, whereas \APCP's asynchronous communication preserves the asynchrony inherent in \MPSTs and real-world distributed systems.
In fact, solving the shortcomings of these prior works was a major motivation for the development of \APCP.

The move to \APCP as modeling system introduces an important challenge.
Connecting processes that model protocol participants in a distributed fashion means connecting each pair of processes directly.
Hence, the local types of \MPSTs, that describe one participant's message exchanges as one protocol, are insufficient as binary session types: we need protocols that describe message exchanges between pairs of participants.

Scalas \etal~\cite{conf/ecoop/ScalasDHY17} introduce another layer of projection, called \emph{partial projection} that projects local types onto binary session types that represent such protocols between pairs of participants.
However, partial projection relies on complex operations (another form of the merge operator discussed in \Cref{s:intro:contrib:mpst:relProj}) to support interesting global protocols.
As it turns out, these multiple layers of projection and merge are incompatible with the approach based on \APCP, so another form of projection is needed.

This is where relative types and relative projection, introduced in \Cref{s:intro:contrib:mpst:relProj}, come in.
By encoding non-local choices in global protocols as dependency messages, relative projection evades a need for operators such as merge while still supporting interesting global types.
The design of relative projection was highly motivated by the analysis in \Cref{c:mpstAPCP}.

\Cref{c:mpstAPCP} then sets up a framework for the analysis of distributed models of \MPSTs in \APCP, relying on relative projection for processes typing.
Were we to use no dependency messages in relative types, this is enough: deadlock-freedom and protocol conformance hold when process models of participants are connected directly to each other.
However, we need to include dependency messages to support interesting protocols, and this requires care.
Dependencies require participants to forward choices received/sent on one channel on other channels, but using only types does not guarantee that the forwarded choices are correct.

\ExecuteMetaData[\fileSCICO]{intro:mpstAPCP:diagram}

The framework's solution is a form of distributed orchestration: it includes orchestrators as in~\cite{conf/forte/CairesP16,conf/concur/CarboneMSY15,conf/concur/CarboneLMSW16,journal/ai/CarboneMSY17}, though not centralized in one process, but distributed among local orchestrators called \emph{routers} at each process.
Routers are processes that are synthesized from global types, using relative projection to detect dependencies.
They forward messages between a participant's process and the rest of the network, forwarding choices as dependency messages when necessary.
\Cref{f:intro:mpstAPCP:diag} illustrates how our approach with distributed routers compares to approaches that use centralized orchestrators.

Through meta-theoretical results, the framework guarantees important correctness properties for model implementations of \MPSTs in \APCP.
First, because routers are well-typed, if participant processes are well-typed as well, then their composition as a distributed system is well-typed as a whole and thus deadlock-free.
Second, such distributed systems conform to the protocol specified by the governing global type.
Third, the framework generalizes the existing centralized solutions~\cite{conf/forte/CairesP16,conf/concur/CarboneMSY15,conf/concur/CarboneLMSW16,journal/ai/CarboneMSY17}, as witnessed by a behavioral equivalence of systems with a centralized orchestrator and with distributed routers.
Finally, \Cref{c:mpstAPCP} includes a series of examples that illustrate the compatibility of the analysis framework with features of \APCP such as interleaving and delegation, under which processes can exchange roles within multiparty protocols.

\subsubsection{Runtime Verification of Blackbox Implementations \texorpdfstring{\\}{} of Multiparty Session Types}
\label{s:intro:contrib:mpst:mpstMon}

Most of the chapters introduced so far are about \emph{static verification}.
That is, these chapters introduce models of systems that can be verified by inspecting their specification: they predict a program's behavior by looking at the actions they specify along with their types.

In practice, having access to the specifications and types of all the components of a system is a rather strong assumption.
Often, systems appeal to third-party components, whose specifications are not publically available, or whose publically available specifications are incomplete or outdated.
Alternatively, components may be specified in a plethora of languages---some typed, some untyped.
This way, developing a unified approach to verification---necessary for the verification of the system as a whole---is a complex task and often very specific to the system under scrutiny.

It is thus essential to develop and research techniques for \emph{dynamic verification}~\cite{conf/rv/ChenR03}.
In dynamic verification, the behavior of programs is inspected \emph{as it is executed}; in other words, the execution of the program is \emph{observed}.
Generally, dynamic verification does not provide as strong guarantees as static verification does, because behavior observed in the past does not give any guarantees about behavior in the future.
Hence, dynamic verification can be seen as a more \emph{passive} form of verification, where the observed behavior is compared against what is expected of a system, and unexpected behavior is reported.
Once unexpected behavior has been observed, one can repair or roll back the behavior on the fly, or it might warrant a (more costly) static inspection of the program to find the source of the flaw.

The dynamic verification of distributed systems is an active field of research (for a survey, see~\cite{book/FrancalanzaPS18}).
In context of my dissertation, it is natural to try to bring \MPSTs into the mix, and this has been done before.
To motivate my contributions, let me briefly discuss the state of the art, in particular the work by Bocchi \etal~\cite{conf/tgc/ChenBDHY11,conf/forte/BocchiCDHY13,journal/tcs/BocchiCDHY17}.

Bocchi \etal's framework considers the dynamic verification of systems consisting of components that implement the roles of participants of global types.
Their global types are more than representations of multiparty protocols, as they include \emph{assertions} about the values that are exchanged.
Dynamic verification is implemented by equipping each component with a \emph{monitor}: a finite state machine (\FSM) that accepts sequences of incoming and outgoing messages; when an unexpected message is observed, the monitor simply drops it.
In their work, Bocchi \etal use local types, projected from the global type, as monitors; hence, they monitor protocol conformance.
Their system then guarantees \emph{safety}---monitored components never behave unexpectedly (as monitors drop unexpected messages)---and \emph{transparency}---monitors do not interfere with (expected) messages.

Bocchi \etal are able to give such strong guarantees because the components they dynamically verify are typed \picalc processes (similar to those in~\cite{conf/popl/HondaYC08,journal/acm/HondaYC16}); that is, they mix dynamic and static verification.
The question is then to what extent their approach is applicable to \emph{blackbox components} (or simply \emph{blackboxes}).
A blackbox is a program whose specification is unknown but whose behavior is observable, e.g., a third-party, closed-source program.
This leads me to my final research question:

\requMpstMon*

My answer to this question is presented in \Cref{c:mpstMon}.
This chapter presents a framework for the dynamic verification of \emph{networks of monitored blackboxes}.
Blackboxes are processes with unknown specification but observable behavior in the form of a labeled transition system (\LTS), with minimal assumptions.
Each blackbox is equipped with a monitor (i.e., a \FSM), forming a monitored blackboxes.
The monitored blackboxes then form a network, in which they communicate asynchronously through buffers.

\ExecuteMetaData[\fileRV]{intro:mpstMon:diagram}

This framework is in many ways a dynamic adaption of the static verification in \Cref{c:mpstAPCP} (introduced in \Cref{s:intro:contrib:mpst:mpstAPCP}).
Blackboxes implement the roles of the participants of global types, and their monitors are synthesized from global types using relative projection, using an algorithm that is very similar to the synthesis of routers discussed in \Cref{c:mpstAPCP}.
Hence, similar to the monitors of Bocchi \etal~\cite{conf/tgc/ChenBDHY11,conf/forte/BocchiCDHY13,journal/tcs/BocchiCDHY17}, these monitors verify the protocol conformance of blackboxes.
\Cref{f:intro:mpstMon:diag} illustrates our approach.

Besides components being blackboxes, the framework in \Cref{c:mpstMon} has another major difference with Bocchi \etal's: instead of ignoring unexpected messages, monitors signal errors when they encounter unexpected messages.
These error signal propagate through the entire network, thus stopping execution upon protocol violations.
A correctness guarantee such as Bocchi \etal's safety (absence of protocol violations) is thus out of the question.

To confirm that a monitored blackbox conforms to its part in the global protocol (i.e., to relative projections from a global type), \Cref{c:mpstMon} defines \emph{satisfaction}.
Satisfaction is a fine-grained, dynamic correctness property, that thoroughly inspects the behavior of a monitored blackbox and compares it to its protocol specification divided among relative types.
Then, when a monitored blackbox satisfies the protocol, we are sure that its monitor will never signal an error (i.e., a protocol violation).

When all monitored blackboxes in a network satisfy their part of the global protocol, the network as a whole conforms to the protocol.
More precisely, any interactions that occur in the network follow the interactions prescribed by the governing global type.
This important correctness guarantee is called \emph{soundness}.

Soundness shows that satisfaction defines a \emph{compositional} method of verifying the protocol conformance of networks of monitored blackboxes.
That is, overal conformance follows from the conformance of individual components, which can be checked in isolation, i.e., without needing to run any other components.
This is especially useful when different parties develop their components separately (and they might not want to share their specifications).

The other property proved in \Cref{c:mpstMon} is transparency, similar to Bocchi \etal's.
Transparency means that monitors interfere minimally with the blackboxes they observe.
Like soundness, transparency leans on satisfaction: clearly, a protocol violation leading to an error signal alters the blackboxes behavior.

As mentioned before, the framework presented in \Cref{c:mpstMon} leverages the relative types and projections introduced in \Cref{s:intro:contrib:mpst:relProj}.
Though dealing with the involved dependency messages complicates satisfaction, soundness and transparency, this approach enables the framework to support a wide range of practical multiparty protocols (different from those supported by Bocchi \etal~\cite{conf/tgc/ChenBDHY11,conf/forte/BocchiCDHY13,journal/tcs/BocchiCDHY17}).

\section{Thesis Outline and Derived Publications}
\label{s:intro:outline}

Starting at \Cpageref{c:APCP}, the contributions in this thesis are organized as follows:

\begin{itemize}

    \item[\Cref{p:pi}] ``\nameref{p:pi}'' (\Cref{s:intro:contrib:pi}).

        \begin{itemize}

            \item[\Cref{c:APCP}] ``\nameref{c:APCP}''.

            \item[\Cref{c:clpi}] ``\nameref{c:clpi}''.

            \item[\Cref{c:piBI}] ``\nameref{c:piBI}''.

        \end{itemize}

    \item[\Cref{p:lambda}] ``\nameref{p:lambda}'' (\Cref{s:intro:contrib:lambda}).

        \begin{itemize}

            \item[\Cref{c:LASTn}] ``\nameref{c:LASTn}''.

            \item[\Cref{c:alphalambda}] ``\nameref{c:alphalambda}''.

        \end{itemize}

    \item[\Cref{p:mpst}] ``\nameref{p:mpst}'' (\Cref{s:intro:contrib:mpst}).

        \begin{itemize}

            \item[\Cref{c:relProj}] ``\nameref{c:relProj}''.

            \item[\Cref{c:mpstAPCP}] ``\nameref{c:mpstAPCP}''.

            \item[\Cref{c:mpstMon}] ``\nameref{c:mpstMon}''.

        \end{itemize}

\end{itemize}
References from all chapters are collected on \Cpageref{refs}.
\Cref{c:basepi} serves as a gentle introduction to the session-typed \picalc, that can be safely skipped by those familiar; each chapter is self-contained.
\ifappendix Starting at \Cpageref{ac:basepi}, for a complete and self-contained dissertation, appendices include detailed definitions and proofs.
\else The extended version of this dissertation~\cite{thesis/vdHeuvel24ex} includes appendices with detailed definitions and proofs. \fi

The chapters in this dissertation are derived from publications, as follows; note that some chapters are derived from multiple publications:

\begingroup
\defcounter{maxnames}{99}
\begin{itemize}

    \item
        \Cref{c:APCP,c:relProj,c:mpstAPCP,c:mpstMon} are derived from the following journal paper (a further development of~\cite{conf/ice/vdHeuvelP21,conf/places/vdHeuvelP20}):
        \\
        \fullcite{journal/scico/vdHeuvelP22}.

    \item
        \Cref{c:APCP,c:LASTn,c:mpstAPCP} are derived from a journal paper with Jorge A.\ Pérez under submission (superseding~\cite{conf/places/vdHeuvelP20}); the following is a preprint:
        \\
        \fullcite{report/vdHeuvelP24}.

    \item
        \Cref{c:clpi} is derived from the following conference paper:
        \\
        \fullcite{conf/aplas/vdHeuvelPNP23}.

    \item
        \Cref{c:piBI,c:alphalambda} are derived from the following conference paper:
        \\
        \fullcite{conf/oopsla/FruminDHP22}.

    \item
        \Cref{c:relProj,c:mpstMon} are derived from the following conference paper:
        \\
        \fullcite{conf/rv/vdHeuvelPD23}.

\end{itemize}
\endgroup


\addtocontents{toc}{\protect\hrulefill\protect\leavevmode\par}
\mypart[Binary Session Types for Message-passing Processes]{Binary Session Types \texorpdfstring{\\}{} for Message-passing Processes}
\label{p:pi}
\chapter{Message-passing Concurrency: The \texorpdfstring{\picalc}{pi-calculus}}
\label{c:basepi}

The chapters presented in this thesis all concern fundamental properties of message-passing concurrency.
To this end, each chapter analyzes mathematical models of message-passing.
A salient approach therein is the \picalc, which (almost) every chapter uses in some form.

The \picalc is a well-known process calculus that focusses on message-passing, originally introduced by Milner \etal~\cite{book/Milner89,journal/ic/MilnerPW92}.
In the \picalc, parallel processes communicate by exchanging messages on channels.
A prominent feature is that such messages can be channels themselves, effectively changing the way processes are connected when channels are exchanged.

There are many variants of the \picalc, each targeting a different specific aspect of message-passing.
In this chapter, I present \basepi, a base variant of the \picalc that is common to all variants of the \picalc in this thesis.
Each chapter then details how \basepi is extended to form a variant of the \picalc that is especially suitable for the topic therein.

An important aspect of \basepi is mode of communication.
A synchronous mode of communication means that a process sending a message can only continue after that message has been received.
In contrast, under an asynchronous mode of communication sending is non-blocking, in that a sending process can continue even before its message has been received~\cite{conf/ecoop/HondaT91,conf/occ/HondaT91,report/Boudol92}.
Some chapters in this thesis rely on synchronous communication, while others require asynchronous communication.
In principle, \basepi features synchronous communication.
However, it has been designed in such a way that asynchronous communication can easily be modelled.
This way, \basepi is a good basis for all variants of the \picalc in this thesis.

In \Cref{s:basepi:syntax,s:basepi:semantics}, I present the syntax and semantics of \basepi, respectively: how processes are constructed and how they behave.
In \Cref{s:basepi:typeSys}, I present a type system for \basepi; types are a common approach to correctness guarantees in this thesis.
Finally, in \Cref{s:basepi:syncAsync}, I discuss how \basepi can be restricted to enforce synchronous or asynchronous communication.

\section{Syntax}
\label{s:basepi:syntax}

The main protagonists of \basepi are \emph{names}, denoted $a,b,\ldots,x,y,\ldots$.
Names represent \emph{endpoints} of channels.
Processes, denoted $P,Q,\ldots$, are defined as sequences of \emph{communications} on names: sends and receives, selections and branches, and closes and waits.
In fact, processes are \emph{parallel compositions} of such sequences.
Processes also include \emph{scoped restrictions} (simply, restrictions), which indicate pairs of names that together form a communication channel (i.e., the names are the opposite endpoints of the channel).
As we will see, connecting names to form channels is what enables communication.

In \basepi, all names are used \emph{linearly}: each name is used for exactly one communication.
Linearity is a useful property when using typing to guarantee correctness properties, as we will see in \Cref{s:basepi:typeSys} and the other chapters of this thesis.

Linear names are nonetheless part of \emph{sessions}, ordered sequences of communications.
In order to implement sessions using linear, one-shot names, \basepi employs \emph{dyadic} communication.
That is, besides a message, outputs carry a \emph{continuation} name; after the message and continuation are received, the session continues on the continuation name.

\begin{definition}[Syntax for \Basepi]
    \label{d:basepi:syntax}
    \begin{align*}
        P, Q &
        \begin{array}[t]{@{}l@{}lr@{\kern2em}l@{\kern1em}lr@{}}
            {} ::= {} &
            \pOut x[a,b] ; P & \text{send}
            & \sepr &
            \pIn x(y,x') ; P & \text{receive}
            \\ \sepr* &
            \pSel x[b] < \ell ; P & \text{select}
            & \sepr &
            \pBra x(x') > {\{ i . P \}_{i \in I}} & \text{branch}
            \\ \sepr* &
            \pClose x[] & \text{close}
            & \sepr &
            \pWait x() ; P & \text{wait}
            \\ \sepr* &
            P \| Q & \text{parallel}
            & \sepr &
            \pRes {xy} P & \text{restriction}
            \\ \sepr* &
            \0 & \text{inaction}
            & \sepr &
            \pFwd [x<>y] & \text{forwarder}
        \end{array}
    \end{align*}
\end{definition}

Let us discuss each constructor in \Cref{d:basepi:syntax}:
\begin{itemize}

    \item
        A send $\pOut x[a,b] ; P$ sends along name $x$ two names $a$ and $b$, and continues as $P$ thereafter.
        A receive $\pIn x(y,x') ; P$ receives along $x$ and continues as $P$; the received names $y$ and $x'$ are placeholder names used in $P$, and they will be substituted for the actual names received.
        A receive $\pIn x(y,x') ; P$ binds $y$ and $x'$ in $P$.

    \item
        A select $\pSel x[b] < \ell ; P$ sends along $x$ some \emph{label} $\ell$ along with a name $b$, and continues as $P$ thereafter; the label represents a choice among a set of alternatives.
        A branch $\pBra x(x') > {\{ i . P_i \}_{i \in I}}$ receives along $x$ a label $i \in I$ along with $x'$, continuing as $P_i$; again, $x'$ is a placeholder that will be substituted for upon reception.
        A branch $\pBra x(x') > {\{ i . P_i \}_{i \in I}}$ binds $x'$ in $P_i$ for every $i \in I$.

    \item
        A close $\pClose x[]$ ends the session on name $x$; it thus has no continuation.
        A wait $\pWait x(); P$ waits for the session on $x$ to close and continues as $P$ thereafter.

    \item
        Parallel composition is denoted $P \| Q$.
        Restriction is denoted $\pRes{xy} P$, connecting the names $x$ and $y$ in $P$.
        Restriction $\pRes{xy} P$ binds $x$ and $y$ in $P$.

    \item
        Inaction $\0$ denotes a process that has nothing to do.

    \item
        A forwarder $\pFwd [x<>y]$ denotes a bidirectional link between names $x$ and $y$, effectively forwarding all inputs on $x$ to $y$ and vice versa.

\end{itemize}
In the following, I often write ``output'' to denote either a send, a select, or a close.
I write $\fn(P)$ to denote the free names of $P$, i.e., the names of $P$ that are not bound.
Substitution of names, denoted $P \{y/x\}$, is defined as replacing all free occurrences of $x$ in $P$ with $y$.
I write $P \{y_1/x_1,\ldots,y_n/x_n\}$ to denote $P \{y_1/x_1\} \ldots \{y_n/x_n\}$.

\section{Semantics}
\label{s:basepi:semantics}

There are several ways to define semantics for the \picalc, for example as a Labeled Transition System.
Here, I present a \emph{reduction semantics} for \basepi.
A reduction step represents two parallel processes performing complementary communications (such as a send and a receive) on connected names; such a step is often referred to as a \emph{synchronization}.

Rules for reduction, given in \Cref{d:basepi:redd}, require specific shapes of processes.
However, processes may not always fit such shapes.
In such cases, the process can be rearranged to fit the shapes required by reduction rules, without affecting the overall behavior of the process.
This rearrangement is called \emph{structural congruence}, meaning that the rules for rearrangement can be applied inside any context induced by the syntax of processes.

\begin{definition}[Structural Congruence \texorpdfstring{($\equiv$)}{} for \Basepi]
    \label{d:basepi:strcong}
    Structural congruence for \basepi, denoted $P \equiv Q$, is defined by the following rules and closed under all contexts induced by the syntax in \Cref{d:basepi:syntax}:
    \begin{mathpar}
        \begin{bussproof}[cong-alpha]
            \bussAssume{
                P \equiv_\alpha Q
            }
            \bussUn{
                P \equiv Q
            }
        \end{bussproof}
        \and
        \begin{bussproof}[cong-par-unit]
            \bussAx{
                P \| \0 \equiv P
            }
        \end{bussproof}
        \and
        \begin{bussproof}[cong-par-comm]
            \bussAx{
                P \| Q \equiv Q \| P
            }
        \end{bussproof}
        \and
        \begin{bussproof}[cong-par-assoc]
            \bussAx{
                P \| (Q \| R) \equiv (P \| Q) \| R
            }
        \end{bussproof}
        \and
        \begin{bussproof}[cong-scope]
            \bussAssume{
                x,y \notin \fn(P)
            }
            \bussUn{
                P \| \pRes{xy} Q \equiv \pRes{xy} ( P \| Q )
            }
        \end{bussproof}
        \and
        \begin{bussproof}[cong-res-comm]
            \bussAx{
                \pRes{xy} \pRes{zw} P \equiv \pRes{zw} \pRes{xy} P
            }
        \end{bussproof}
        \and
        \begin{bussproof}[cong-res-symm]
            \bussAx{
                \pRes{xy} P \equiv \pRes{yx} P
            }
        \end{bussproof}
        \and
        \begin{bussproof}[cong-fwd-symm]
            \bussAx{
                \pFwd [x<>y] \equiv \pFwd [y<>x]
            }
        \end{bussproof}
        \and
        \begin{bussproof}[cong-res-fwd]
            \bussAx{
                \pRes{xy} \pFwd [x<>y] \equiv \0
            }
        \end{bussproof}
    \end{mathpar}
\end{definition}

Let us discuss each rule in \Cref{d:basepi:strcong}:
\begin{itemize}

    \item
        Rule~\ruleLabel{cong-alpha} defines $\alpha$-equivalent processes (i.e., processes that are equal up to renaming of bound names) as structurally congruent.

    \item
        Rule~\ruleLabel{cong-par-unit} defines $\0$ as the unit of parallel composition.
        Rules~\ruleLabel{cong-par-comm} and~\ruleLabel{cong-par-assoc} define parallel composition as commutative and associative, respectively.

    \item
        Rule~\ruleLabel{cong-scope} allows extending the scope of a restriction, as long as this does not bind any free names.
        Rules~\ruleLabel{cong-res-comm} and~\ruleLabel{cong-res-symm} define restriction as commutative and symmetric, respectively.

    \item
        Rule~\ruleLabel{cong-fwd-symm} defines forwarders as symmetric.
        Rule~\ruleLabel{cong-res-fwd} says that a forwarder with both names bound together is the same as inaction.

\end{itemize}

I now define reduction:

\begin{definition}[Reduction \texorpdfstring{($\protect\redd$)}{} for \Basepi]
    \label{d:basepi:redd}
    Reduction for \basepi is a relation between processes, denoted $P \redd Q$.
    It is defined by the following rules:
    \begin{mathpar}
        \begin{bussproof}[red-send-recv]
            \bussAx{
                \pRes{xy} ( \pOut x[a,b] ; P \| \pIn y(z,y') ; Q )
                \redd
                P \| Q \{a/z,b/y'\}
            }
        \end{bussproof}
        \and
        \begin{bussproof}[red-sel-bra]
            \bussAssume{
                j \in I
            }
            \bussUn{
                \pRes{xy} ( \pSel x[b] < j ; P \| \pBra y(y') > {\{ i . Q_i \}_{i \in I}} )
                \redd
                P \| Q_j \{b/y'\}
            }
        \end{bussproof}
        \and
        \begin{bussproof}[red-close-wait]
            \bussAx{
                \pRes{xy} ( \pClose x[] \| \pWait y() ; Q )
                \redd
                Q
            }
        \end{bussproof}
        \and
        \begin{bussproof}[red-fwd]
            \bussAssume{
                x,y \neq z
            }
            \bussUn{
                \pRes{xy}( \pFwd [x<>z] \| Q )
                \redd
                Q \{z/y\}
            }
        \end{bussproof}
        \and
        \begin{bussproof}[red-cong]
            \bussAssume{
                P \equiv P'
            }
            \bussAssume{
                P' \redd Q'
            }
            \bussAssume{
                Q' \equiv Q
            }
            \bussTern{
                P \redd Q
            }
        \end{bussproof}
        \and
        \begin{bussproof}[red-res]
            \bussAssume{
                P \redd Q
            }
            \bussUn{
                \pRes{xy} P \redd \pRes{xy} Q
            }
        \end{bussproof}
        \and
        \begin{bussproof}[red-par]
            \bussAssume{
                P \redd Q
            }
            \bussUn{
                P \| R \redd Q \| R
            }
        \end{bussproof}
    \end{mathpar}
    I write $\redd*$ to denote the reflexive, transitive closure of $\redd$.
\end{definition}

Let us discuss each rule in \Cref{d:basepi:redd}:
\begin{itemize}

    \item
        Rule~\ruleLabel{red-send-recv} synchronizes a send and a receive on names connected by restriction.
        Rule~\ruleLabel{red-sel-bra} synchronizes a selection and a branch on similarly connected names.
        Rule~\ruleLabel{red-close-wait} similarly synchronizes a close and a wait.
        In each case, the received placeholder names are substituted for the sent names.
        Moreover, the restriction is removed, as the involved names will no longer be used (as per linearity).

    \item
        Rule~\ruleLabel{red-fwd} ``shortcuts'' a forwarder when one of its names is connected to a parallel process by restriction.
        It effectively substitutes the connected name for the forwarder's other name.

    \item
        Rules~\ruleLabel{red-cong}, \ruleLabel{red-res}, and~\ruleLabel{red-par} close reduction under structural congruence, restriction, and parallel composition, respectively.

\end{itemize}

\section{Session Type System}
\label{s:basepi:typeSys}

As mentioned in \Cref{s:basepi:syntax}, the linear names of \basepi are part of sessions, ordered sequences of communications.
It is then a natural choice to statically analyze \basepi processes using \emph{session types}.
Session types, introduced originally by Honda \etal~\cite{conf/concur/Honda93,conf/esop/HondaVK98,conf/parle/TakeuchiHK94}, represent communication protocols on names as sequences of communications.

Every chapter in this thesis uses session types in some way.
Many of these chapters are inspired by a Curry-Howard correspondence between linear logic~\cite{journal/tcs/Girard87} and session types for the \picalc, introduced by Caires and Pfenning~\cite{conf/concur/CairesP10} and Wadler~\cite{conf/icfp/Wadler12}.
In summary, the correspondence connects session types and linear logic propositions, typing derivations and sequent calculus derivations, and reduction semantics and cut-reduction.
Because linear logic is a recurring theme in this thesis, I present session types in the form of linear logic propositions.
On the other hand, the type system I present here does not correspond to a linear logic sequent calculus, but rather is closer to the original session type systems for the \picalc by Honda and others~\cite{conf/concur/Honda93,conf/esop/HondaVK98,conf/parle/TakeuchiHK94}.
I will discuss the Curry-Howard correspondence between linear logic and session types in more detail in \Cref{p:pi}.

The general idea of the session type system for \basepi I present here is that the free names of a process are assigned session types.
These session types then prescribe how the process should behave on those names.
Session types for \basepi, in the form of linear logic propositions, are defined as follows:

\begin{definition}[Session Types for \Basepi]
    \label{d:basepi:types}
    \begin{align*}
        A,B &
        \begin{array}[t]{@{}l@{}lr@{\kern2em}l@{\kern1em}lr@{}}
            {} ::= {}&
            A \tensor B & \text{(send)}
            & \sepr &
            A \parr B & \text{(receive)}
            \\ \sepr* &
            \oplus_{i \in I} A & \text{(select)}
            & \sepr &
            \&_{i \in I} A & \text{(branch)}
            \\ \sepr* &
            \1 & \text{(close)}
            & \sepr &
            \bot & \text{(wait)}
        \end{array}
    \end{align*}
\end{definition}

Let us discuss the constructors in \Cref{d:basepi:types}:
\begin{itemize}

    \item
        The type $A \tensor B$ is assigned to a name on which a name of type $A$ is sent, along with a continuation name of type $B$.
        Dually, the type $A \parr B$ is assigned to a name on which a name of type $A$ is received, along with a continuation name of type $B$.
        As we will see, these types are assigned to processes that send and receive, respectively.

    \item
        The type $\oplus_{i \in I} A_i$ is assigned to a name on which a label $i \in I$ is sent, along with a continuation name of type $A_i$.
        Dually, the type $\&_{i \in I} A_i$ is assigned to a name on which a label $i \in I$ is received, along with a continuation name of type $A_i$.
        As we will see, these types are assigned to processes that select and branch, respectively.

    \item
        The type $\1$ is assigned to a name on which the session is closed.
        Dually, the type $\bot$ is assigned to a name on which the closing of the session is awaited.
        As we will see, these types are assigned to processes that close and wait, respectively.

\end{itemize}
Additionally, each of these types may also be assigned to forwarders.

The description above mentions \emph{dual} types.
Duality is the cornerstone of session types.
The idea is that opposite names of a channel should be assigned dual types.
This makes sure that the names are not simultaneously used for conflicting communications, e.g., sending on both names at the same time.
Duality is defined as follows:

\begin{definition}[Duality]
    \label{d:basepi:dual}
    Duality for session types, denoted $\ol{A}$, is defined as follows:
    \begin{align*}
        \ol{\1} &\deq \bot
        &
        \ol{A \tensor B} &\deq \ol{A} \parr \ol{B}
        &
        \ol{\oplus_{i \in I} A_i} &\deq \&_{i \in I} \ol{A_i}
        \\
        \ol{\bot} &\deq \1
        &
        \ol{A \parr B} &\deq \ol{A} \tensor \ol{B}
        &
        \ol{\&_{i \in I} A_i} &\deq \oplus_{i \in I} \ol{A_i}
    \end{align*}
\end{definition}

\noindent
Clearly, duality is an involution: $\ol{\;\ol{A}\;} = A$ for any session type $A$.

To type a process, each of its free names is assigned a session type, denoted $x:A$.
A process is then assigned a typing context $\Delta$ consisting of such assignments, denoted as the typing judgment $\vdash P \typInf \Delta$.
Writing $x:A , y:B$ implicitly assumes that $x \neq y$, which straightforwardly extends to typing contexts.
Typing is implicitly closed under exchange, i.e., $\Delta_1 , x:A , y:B , \Delta_2 = \Delta_1 , y:B , x:A , \Delta_2$.
The type system for \basepi is then defined by the following rules:

\begin{definition}[Type System for \Basepi]
    \label{d:basepi:typeSys}
    \begin{mathpar}
        \begin{bussproof}[typ-send]
            \bussAssume{
                \vdash P \typInf \Delta
            }
            \bussUn{
                \vdash \pOut x[a,b] ; P \typInf \Delta , x:A \tensor B , a:\ol{A} , b:\ol{B}
            }
        \end{bussproof}
        \and
        \begin{bussproof}[typ-recv]
            \bussAssume{
                \vdash P \typInf \Delta , y:A , x':B
            }
            \bussUn{
                \vdash \pIn x(y,x') ; P \typInf \Delta , x:A \parr B
            }
        \end{bussproof}
        \and
        \begin{bussproof}[typ-sel]
            \bussAssume{
                \vdash P \typInf \Delta
            }
            \bussAssume{
                j \in I
            }
            \bussBin{
                \vdash \pSel x[b] < j ; P \typInf \Delta , x:\oplus_{i \in I} A_i , b:\ol{A_j}
            }
        \end{bussproof}
        \and
        \begin{bussproof}[typ-bra]
            \bussAssume{
                \forall i \in I.~
                \vdash P_i \typInf \Delta , x':A_i
            }
            \bussUn{
                \vdash \pBra x(x') > {\{ i . P_i \}_{i \in I}} \typInf \Delta , x:\&_{i \in I} A_i
            }
        \end{bussproof}
        \and
        \begin{bussproof}[typ-close]
            \bussAx{
                \vdash \pClose x[] \typInf x:\1
            }
        \end{bussproof}
        \and
        \begin{bussproof}[typ-wait]
            \bussAssume{
                \vdash P \typInf \Delta
            }
            \bussUn{
                \vdash \pWait x() ; P \typInf \Delta , x:\bot
            }
        \end{bussproof}
        \and
        \begin{bussproof}[typ-par]
            \bussAssume{
                \vdash P \typInf \Delta_1
            }
            \bussAssume{
                \vdash Q \typInf \Delta_2
            }
            \bussBin{
                \vdash P \| Q \typInf \Delta_1 , \Delta_2
            }
        \end{bussproof}
        \and
        \begin{bussproof}[typ-res]
            \bussAssume{
                \vdash P \typInf \Delta , x:A , y:\ol{A}
            }
            \bussUn{
                \vdash \pRes{xy} P \typInf \Delta
            }
        \end{bussproof}
        \and
        \begin{bussproof}[typ-inact]
            \bussAx{
                \vdash \0 \typInf \emptyset
            }
        \end{bussproof}
        \and
        \begin{bussproof}[typ-fwd]
            \bussAx{
                \vdash \pFwd [x<>y] \typInf x:A , y:\ol{A}
            }
        \end{bussproof}
    \end{mathpar}
\end{definition}

Let us discuss each rule in \Cref{d:basepi:typeSys}:
\begin{itemize}

    \item
        Rule~\ruleLabel{typ-send} types a send $\pOut x[a,b] ; P$ by assigning to $x$ the type $A \tensor B$, to $a$ the type~$\ol{A}$, and to $b$ the type $\ol{B}$; $P$ may be arbitrarily types.
        It may seem counterintuitive that the types assigned to $a$ and $b$ are dual to the types assigned to $x$.
        However, consider that $\ol{A \tensor B} = \ol{A} \parr \ol{B}$: the receiving name of the channel of which $x$ is part is expecting to receive names of types $\ol{A}$ and $\ol{B}$.

    \item
        Rule~\ruleLabel{typ-recv} types a receive $\pIn x(y,x') ; P$ by assigning to $x$ the type $A \parr B$.
        It requires that in $P$ the name $y$ is assigned the type $A$, and $x'$ the type $B$.

    \item
        Rule~\ruleLabel{typ-sel} types a selection $\pSel x[b] < j ; P$ by assigning to $x$ the type $\oplus_{i \in I} A_i$ with \mbox{$j \in I$}, and to $b$ the type $\ol{A_j}$; $P$ may be arbitrarily typed.
        Rule~\ruleLabel{typ-bra} types a branch \mbox{$\pBra x(x') > {\{ i . P_i \}_{i \in I}}$} by assigning to $x$ the type $\&_{i \in I} A_i$.
        It requires that, for every $i \in I$, in $P_i$ the name $x'$ is assigned the type $A_i$.

    \item
        Rule~\ruleLabel{typ-close} types a close $\pClose x[]$ by assigning to $x$ the type $\1$; no other assignments may appear in the typing context.
        Rule~\ruleLabel{typ-wait} types a wait $\pWait x() ; P$ by assigning to $x$ the type $\bot$; $P$ may be arbitrarily types.

    \item
        Rule~\ruleLabel{typ-par} types the parallel composition of two typed processes.
        Rule~\ruleLabel{typ-res} types a restriction $\pRes{xy} P$ by requiring that the names $x$ and $y$ are assigned dual types.

    \item
        Rule~\ruleLabel{typ-inact} types inaction $\0$ with an empty typing context.
        Rule~\ruleLabel{typ-fwd} types a forwarder $\pFwd [x<>y]$ by assigning to $x$ and $y$ dual types.

\end{itemize}

It may not seem entirely clear how the rules in \Cref{d:basepi:typeSys} type sessions, i.e., ordered sequences of communications.
This will become clearer in \Cref{s:basepi:syncAsync}, where I discuss synchronous versus asynchronous modes of communication that determine the method of implementing session using the rules in \Cref{d:basepi:typeSys}.

The foundation of correctness guarantees from type systems is \emph{type preservation}.
Type preservation means that typing is preserved as processes transform, under structural congruences (\Cref{t:basepi:subjCong}) and under reduction (\Cref{t:basepi:subjRed}).
A direct consequence of type preservation is that well-typed \basepi processes satisfy \emph{protocol fidelity} (they correctly implement their assigned session types) and \emph{communication safety} (there are no message mismatches).
Proofs are detailed in\ifappendix \Cref{ac:basepi}\else~\cite{thesis/vdHeuvel24ex}\fi.

\begin{restatable}[Subject Congruence]{theorem}{tPiCalcSubjCong}
    \label{t:basepi:subjCong}
    If $\vdash P \typInf \Delta$ and $P \equiv Q$, then $\vdash Q \typInf \Delta$.
\end{restatable}

\begin{sketch}
    By induction on the derivation of $P \equiv Q$.
    The inductive cases, which apply congruence under contexts, follow from the IH straightforwardly.
    The base cases correspond to the axioms in \Cref{d:basepi:strcong}.
    In each case, apply inversion to derive the typing of subprocesses and consequently the typing of the structurally congruent process.
\end{sketch}

\begin{restatable}[Subject Reduction]{theorem}{tPiCalcSubjRedd}
    \label{t:basepi:subjRed}
    If $\vdash P \typInf \Delta$ and $P \redd Q$, then $\vdash Q \typInf \Delta$.
\end{restatable}

\begin{sketch}
    By induction on the derivation of $P \redd Q$.
    The cases correspond to the rules in \Cref{d:basepi:redd}.
    In each case, apply inversion to derive the typing of subprocesses and consequently the typing of the reduced process, in some cases appealing to the IH and \Cref{t:basepi:subjCong}.
\end{sketch}

Another important correctness aspect that is desirable to be guaranteed by typing is \emph{deadlock-freedom} (processes do not get stuck waiting to receive messages from each other).
\Basepi does not guarantee this property.

\begin{example}
    \label{x:basepi:deadlock}
    Consider the following process:
    \begin{align*}
        Q_1 &\deq \pRes{ab} \pOut x[a,b] ; \pRes{cd} \pOut z[c,d] ; \0
        \\
        Q_2 &\deq \pIn w(v,w') ; \pIn y(u,y') ; ( \pFwd [v<>w'] \| \pFwd [u<>y'] )
        \\
        P &\deq \pRes{xy} \pRes{zw} ( Q_1 \| Q_2 )
    \end{align*}
    Clearly, $P$ is deadlocked.
    $Q_1$ tries to send on $x$, but the corresponding receive on $y$ in $Q_2$ is blocked by the receive on $w$.
    In turn, this receive on $w$ is connected to a send on $z$, which is blocked by the send on $x$ in $Q_1$.

    Nonetheless, $P$ is well-typed:
    \begin{mathpar}
        \begin{bussproof}
            \bussAx{
                \vdash \0 \typInf \emptyset
            }
            \bussUn{
                \vdash \pOut z[c,d] ; \0 \typInf z:B \tensor \ol{B} , c:\ol{B} , d:B
            }
            \bussUn{
                \vdash \pRes{cd} \pOut z[c,d] ; \0 \typInf z:B \tensor \ol{B}
            }
            \bussUn{
                \vdash \pOut x[a,b] ; \pRes{cd} \pOut z[c,d] ; \0 \typInf x:A \tensor \ol{A} , a:\ol{A} , b:A , z:B \tensor \ol{B}
            }
            \bussUn{
                \vdash \underbrace{ \pRes{ab} \pOut x[a,b] ; \pRes{cd} \pOut z[c,d] ; \0 }_{Q_1} \typInf x:A \tensor \ol{A} , z:B \tensor \ol{B}
            }
        \end{bussproof}
        \\
        \begin{bussproof}
            \bussAx{
                \vdash \pFwd [v<>w'] \typInf v:\ol{B} , w':B
            }
            \bussAx{
                \vdash \pFwd [u<>y'] \typInf u:\ol{A} , y':A
            }
            \bussBin{
                \vdash \pFwd [v<>w'] \| \pFwd [u<>y'] \typInf v:\ol{B} , w':B , u:\ol{A} , y':A
            }
            \bussUn{
                \vdash \pIn y(u,y') ; ( \pFwd [v<>w'] \| \pFwd [u<>y'] ) \typInf y:\ol{A} \parr A , v:\ol{B} , w':B
            }
            \bussUn{
                \vdash \underbrace{ \pIn w(v,w') ; \pIn y(u,y') ; ( \pFwd [v<>w'] \| \pFwd [u<>y'] ) }_{Q_2} \typInf w:\ol{B} \parr B , y:\ol{A} \parr A
            }
        \end{bussproof}
        \\
        \begin{bussproof}
            \bussAssume{
                \vdash Q_1 \typInf x:A \tensor \ol{A} , z:B \tensor \ol{B}
            }
            \bussAssume{
                \vdash Q_2 \typInf w:\ol{B} \parr B , y:\ol{A} \parr A
            }
            \bussBin{
                \vdash Q_1 \| Q_2 \typInf x:A \tensor \ol{A} , z:B \tensor \ol{B} , w:\ol{B} \parr B , y:\ol{A} \parr A
            }
            \bussUn{
                \vdash \pRes{zw} ( Q_1 \| Q_2 ) \typInf x:A \tensor \ol{A} , y:\ol{A} \parr A
            }
            \bussUn{
                \vdash \pRes{xy} \pRes{zw} ( Q_1 \| Q_2 ) \typInf \emptyset
            }
        \end{bussproof}
    \end{mathpar}
\end{example}

\noindent
The chapters in \Cref{p:pi} extend and restrict \basepi in such a way that deadlock-freedom is guaranteed by typing.

\section[Synchronous versus Asynchronous Communication]{Synchronous versus \texorpdfstring{\\}{} Asynchronous Communication}
\label{s:basepi:syncAsync}

A mentioned in the introduction to this chapter, the mode of communication of \basepi is synchronous: when a process sends a message, it needs to wait for the message to be received.
Typically, the continuation of an output's session is implemented in the output's continuation process.
\Basepi supports this style of implementing sessions, though it does not enforce it:
\begin{align*}
    & \vdash \pRes{bx'} \pOut x[a,b] ; P \typInf \Delta , x:A \tensor B , a:\ol{A}
    & & \text{where } \vdash P \typInf \Delta , x':B
    \\
    & \vdash \pRes{bx'} \pSel x[b] < j ; P \typInf \Delta , x:\oplus_{i \in I} A_i
    & & \text{where } \vdash P \typInf \Delta , x':A_j
\end{align*}

Because \basepi does not enforce that the continuations of outputs' sessions are implemented in their continuation processes, \basepi straightforwardly supports asynchronous communication.
The idea is that an output session's continuation is implemented in parallel to the output, bound to the output's continuation name by restriction:
\begin{align*}
    & \vdash \pRes{bx'} ( \pOut x[a,b] ; \0 \| P ) \typInf \Delta, x:A \tensor B , a:\ol{A}
    & & \text{where } \vdash P \typInf \Delta , x':B
    \\
    & \vdash \pRes{bx'} ( \pSel x[b] < j ; \0 \| P ) \typInf \Delta , x:\oplus_{i \in I}
    & & \text{where } \vdash P \typInf \Delta , x':A_j
\end{align*}

To enforce asynchronous communication, one can restrict the continuations of outputs to always be $\0$.
Omitting the trailing $\0$ of outputs, I can then derive the following typing rules for asynchronous output:
\begin{align*}
    \begin{bussproof}[typ-send-async]
        \bussAx{
            \vdash \pOut x[a,b] \vphantom{; \0} \typInf x:A \tensor B , a:\ol{A} , b:\ol{B}
        }
    \end{bussproof}
    \quad &\deq \quad
    \begin{bussproof}
        \bussAx{
            \vdash \0 \typInf \emptyset
        }
        \bussUn{
            \vdash \pOut x[a,b] ; \0 \typInf x:A \tensor B , a:\ol{A} , b:\ol{B}
        }
    \end{bussproof}
    \\
    \begin{bussproof}[typ-sel-async]
        \bussAssume{
            j \in I
        }
        \bussUn{
            \vdash \pSel x[b] < j \typInf x:\oplus_{i \in I} A_i , b:\ol{A_j}
        }
    \end{bussproof}
    \quad &\deq \quad
    \begin{bussproof}
        \bussAx{
            \vdash \0 \typInf \emptyset
        }
        \bussAssume{
            j \in I
        }
        \bussBin{
            \vdash \pSel x[b] < j ; \0 \typInf x:\oplus_{i \in I} A_i , b:\ol{A_j}
        }
    \end{bussproof}
\end{align*}

Roughly speaking, this approach to sessions with asynchronous communication is equivalent to extending \basepi with ordered message buffers~\cite{conf/csl/DeYoungCPT12}: outputs place their messages in buffers to unblock their continuations, and inputs will retrieve messages from buffers.

\section{Extensions of \texorpdfstring{\Basepi}{Base-pi}}
\label{s:basepi:extensions}

Here, I give a brief overview of how the rest of the chapters in \Cref{p:pi} extend \basepi.
\begin{itemize}
    \item[\Cref{c:APCP}] introduces \APCP.
        Its process language restricts to asynchronous communication (as described in \Cref{s:basepi:syncAsync}), makes the closing of sessions implicit, and adds tail-recursion.
        Its type system conflates \basepi's $\1$ and $\bot$ to $\bullet$, and adds tail-recursive types.
        To guarantee deadlock-freedom, \APCP adds priority annotations to its types and type system.
        \Cref{s:APCP:APCP:exts} additionally describes how to integrate explicit closing (with types $\1$ and~$\bot$) and servers and clients (i.e., unrestricted sessions to be used any number of types).

    \item[\Cref{c:clpi}] introduces \clpi.
        Its process language restricts to synchronous communication, and adds constructs for non-determinism: non-deterministic choice and prefixes for session that may not be available.
        Its type system only adds a typing rule for non-deterministic choice, and modalities for sessions that may not be available.
        To guarantee deadlock-freedom, \clpi replaces Rule~\ruleLabel{typ-res} with linear logic's Rule~\ruleLabel{cut}: this forbids cyclically connected processes, thus preventing circular dependencies and hence deadlocks.

    \item[\Cref{c:piBI}] introduces \pibi.
        Its process language restricts to synchronous communication, and adds the spawn construct.
        Its type system is based on the logic of bunched implications (\BI), rather than linear logic.
        Besides significant differences in structural rules (i.e., how resource duplication and discarding is handled), the type sytem of \pibi uses \emph{two-sided sequents} derived from \BI.
        The types in \pibi are similar to those in \basepi, but each connective has a left- and a right-rule, describing how to \emph{use} and \emph{provide} behavior, respectively.
        The \pibi deadlock-freedom guarantee follows that of \clpi described above, as \BI also has a Rule~\ruleLabel{cut}.
\end{itemize}


\APCPtrue \chapter{Cyclic Networks and Asynchronous Communication}
\label{c:APCP}

This chapter studies new ways of combining message-passing processes with cyclic connections and asynchronous communication.
The setting is that of session type systems derived from linear logic, and a main aim is to maintain deadlock-freedom by typing.
It thus answers the following research question, introduced in \Cref{s:intro:contrib:pi:APCP}:

\requAPCP*

\section{Introduction}
\label{s:APCP:intro}

\ExecuteMetaData[\fileLMCS]{APCP:intro:1}

\ExecuteMetaData[\fileLMCS]{APCP:intro:2}

\ExecuteMetaData[\fileLMCS]{APCP:intro:3}

\section{Example: Milner's Cyclic Scheduler in \texorpdfstring{\APCP}{APCP}}
\label{s:APCP:milner}

\ExecuteMetaData[\fileLMCS]{APCP:milner}

\section[\texorpdfstring{\APCP}{APCP}: Asynchronous Priority-based Classical Processes]{\texorpdfstring{\APCP}{APCP}: Asynchronous Priority-based \texorpdfstring{\\}{} Classical Processes}
\label{s:APCP:APCP}

\ExecuteMetaData[\fileLMCS]{APCP:APCP}

\subsection{The Process Language}
\label{s:APCP:APCP:procLang}

\ExecuteMetaData[\fileLMCS]{APCP:APCP:proclang}

\subsection{The Type System}
\label{s:APCP:APCP:types}

\ExecuteMetaData[\fileLMCS]{APCP:APCP:types}

\subsection{Type Preservation and Deadlock-freedom}
\label{s:APCP:APCP:results}

\ExecuteMetaData[\fileLMCS]{APCP:APCP:results}

\subsection{Reactivity}
\label{s:APCP:APCP:react}

\ExecuteMetaData[\fileLMCS]{APCP:APCP:react}

\subsection{Typing Milner's Cyclic Scheduler}
\label{s:APCP:APCP:milner}

\ExecuteMetaData[\fileLMCS]{APCP:APCP:milner}

\subsection{Extensions: Explicit Session Closing and Replicated Servers}
\label{s:APCP:APCP:exts}

\ExecuteMetaData[\fileLMCS]{APCP:APCP:exts}

\section{Related Work}
\label{s:APCP:rw}

\ExecuteMetaData[\fileLMCS]{APCP:rw}

\section{Closing Remarks}
\label{s:APCP:concl}

\ExecuteMetaData[\fileLMCS]{APCP:concl}

 \APCPfalse
\chapter{Non-determinism}
\label{c:clpi}

This chapter studies non-determinism in message-passing processes.
Though this is not a new theme in general, it has not been thoroughly studied in context of session type systems derived from linear logic.
The aim is to introduce expressive forms of non-deterministic choice, while retaining correctness guarantees inherited from linear logic such as deadlock-freedom.
Hence, this chapter attempts to answer the following research question, introduced in \Cref{s:intro:contrib:pi:clpi}:

\requND*

This chapter is the result of a collaboration between myself and Joseph W.\ N.\ Paulus, Daniele Nantes-Sobrinho, and Jorge A.\ Pérez.
Here, I focus on the parts of our publication (to appear) to which I contributed most significantly: the design of \clpi, its two semantics, and its meta-theoretical results (type preservation and deadlock-freedom for both semantics, and a formal comparison of the semantics).
The publication contains additional contributions: a resource \lamcalc where elements are fetched from bags non-deterministically, a translation from this calculus into \clpi, and operational correspondence results for this translation under both semantics; these contributions are mainly due to Joseph W.\ N.\ Paulus, who will include them in their own dissertation.

\section{Introduction}
\label{s:clpi:intro}

\ExecuteMetaData[\fileAPLAS]{clpi:intro}

\section{A Typed \texorpdfstring{\picalc}{pi-calculus} with Non-deterministic Choice}
\label{s:clpi:clpi}

\ExecuteMetaData[\fileAPLAS]{clpi:clpi}

\subsection{Syntax and Semantics}
\label{s:clpi:clpi:syntaxSemantics}

\ExecuteMetaData[\fileAPLAS]{clpi:clpi:syntaxSemantics}

\subsection{Resource control for \texorpdfstring{\clpi}{clpi} via Session Types}
\label{s:clpi:clpi:typeSys}

\ExecuteMetaData[\fileAPLAS]{clpi:clpi:typeSys}

\subsection{Correctness Properties}
\label{s:clpi:clpi:props}

\ExecuteMetaData[\fileAPLAS]{clpi:clpi:props}

\section{An Eager Semantics for \texorpdfstring{\clpi}{clpi}}
\label{s:clpi:eager}

\ExecuteMetaData[\fileAPLAS]{clpi:eager}

\subsection{The Semantics}
\label{s:clpi:eager:semantics}

\ExecuteMetaData[\fileAPLAS]{clpi:eager:semantics}

\subsection{Comparing Lazy and Eager Semantics}
\label{s:clpi:eager:comp}

\ExecuteMetaData[\fileAPLAS]{clpi:eager:comp}

\section{Related Work}
\label{s:clpi:rw}

\ExecuteMetaData[\fileAPLAS]{clpi:rw}

\section{Conclusions}
\label{s:clpi:concl}

\ExecuteMetaData[\fileAPLAS]{clpi:concl}


\piBItrue \chapter{A Bunch of Sessions}
\label{c:piBI}

This chapter studies alternative logical foundations for session-typed message-passing processes.
In particular, it proposes a Curry-Howard interpretation of the logic of bunched implications as a session type system.
The aim is to reconcile the message-passing features obtained from this interpretation with practical concerns, as well as maintaining correctness guarantees such as deadlock-freedom.
This chapter thus addresses the following research question, introduced in \Cref{s:intro:contrib:pi:piBI}:

\requBI*

This chapter is the result of a collaboration between myself and Dan Frumin, Emanuele D'Osualdo, and Jorge A.\ Pérez.
Here, I focus on the parts of our publication~\cite{conf/oopsla/FruminDHP22} to which I contributed most significantly: the design of \piBI and its meta-theoretical results (type preservation, deadlock-freedom, and weak normalization).
The publication contains additional contributions: a denotational semantics to derive behavior equivalence and provenance tracking (omitted as a whole), and a translation from the \alamcalc into \piBI (\Cref{c:alphalambda}).

\section{Introduction}
\label{s:piBI:intro}

\ExecuteMetaData[\fileOOPSLA{intro}]{piBI:intro}

\section{The Calculus \texorpdfstring{\pibi}{piBI}}
\label{s:piBI:piBI}

\ExecuteMetaData[\fileOOPSLA{calculus}]{piBI:piBI}

\subsection{Process Syntax}
\label{s:piBI:piBI:syntax}

\ExecuteMetaData[\fileOOPSLA{calculus}]{piBI:piBI:syntax}

\subsection{Reduction Semantics}
\label{s:piBI:piBI:reduction}

\ExecuteMetaData[\fileOOPSLA{calculus}]{piBI:piBI:reduction}

\subsection{Typing}
\label{s:piBI:piBI:typing}

\ExecuteMetaData[\fileOOPSLA{calculus}]{piBI:piBI:typing}

\subsection{Examples and Comparisons}
\label{s:piBI:piBI:examples}

\ExecuteMetaData[\fileOOPSLA{calculus}]{piBI:piBI:examples}

\section{Meta-theoretical Properties}
\label{s:piBI:meta}

\ExecuteMetaData[\fileOOPSLA{meta}]{piBI:meta}

\subsection{Type Preservation and Deadlock-freedom}
\label{s:piBI:meta:tpdf}

\ExecuteMetaData[\fileOOPSLA{meta}]{piBI:meta:tpdf}

\subsection{Weak Normalization}
\label{s:piBI:meta:wn}

\ExecuteMetaData[\fileOOPSLA{meta}]{piBI:meta:wn}

\section{Related Work}
\label{s:piBI:rw}

\ExecuteMetaData[\fileOOPSLA{relwork}]{piBI:rw}

\section{Concluding Remarks and Future Perspectives}
\label{s:piBI:concl}

\ExecuteMetaData[\fileOOPSLA{conclusion}]{piBI:concl}

 \piBIfalse

\addtocontents{toc}{\protect\hrulefill\protect\leavevmode\par}
\mypart[Message-passing Functions as Processes]{Message-passing Functions \texorpdfstring{\\}{} as Processes}
\label{p:lambda}

\LASTntrue 
\chapter[Cyclic Thread Configurations and Asynchronous Communication]{Cyclic Thread Configurations \texorpdfstring{\\}{} and \texorpdfstring{\\}{} Asynchronous Communication}
\label{c:LASTn}

This chapter studies message-passing in functional programming.
The focus is on cyclic connections between threads and asynchronous communication.
The aim is to exploit prior knowledge on these aspects on the level of message-passing processes in \APCP.
By building strong ties between these worlds of functional programming and processes, this chapter transfers correctness results from the latter to the former.
It thus provides an answer to the following research question, introduced in \Cref{s:intro:contrib:lambda:LASTn}:

\requLASTn*

\ifappendix
This chapter includes a self-contained summary of \APCP (\Cref{s:LASTn:APCP}).
Omitted details and proofs can be found in the chapter dedicated to \APCP (\Cref{c:APCP}).
\else
This chapter relies on the definition and results of \APCP, presented in detail in \Cref{c:APCP}; because this chapter does not require \APCP's recursion, it is omitted entirely.
\fi

\section{Introduction}
\label{s:LASTn:intro}

Society relies heavily on software consisting of distributed components that cooperate by asynchronously exchanging messages.
It is vital that such software functions as intended and without error.
In this context, we study \emph{session types}: behavioral types that represent communication protocols used in the static verification of message-passing software.

\ExecuteMetaData[\fileLMCS]{LASTn:intro:1}
\ExecuteMetaData[\fileLMCS]{LASTn:intro:2}
\ExecuteMetaData[\fileLMCS]{LASTn:intro:3}

\section{A Bookshop Scenario in \LASTn}
\label{s:LASTn:example}

\ExecuteMetaData[\fileLMCS]{LASTn:example}

\section{An Intermezzo: From \texorpdfstring{\APCP}{APCP} to \texorpdfstring{\protect\LAST}{LAST}}
\label{s:LASTn:LAST}

\ExecuteMetaData[\fileLMCS]{LASTn:LAST}

\subsection{The Syntax and Semantics of \protect\LAST*}
\label{s:LASTn:LAST:syntaxSemantics}

\ExecuteMetaData[\fileLMCS]{LASTn:LAST:syntaxSemantics}

\subsection{The Type System of \protect\LAST*}
\label{s:LASTn:LAST:typeSystem}

\ExecuteMetaData[\fileLMCS]{LASTn:LAST:typeSystem}

\ifappendix
\subsection{\APCP: a Summary}
\label{s:LASTn:APCP}

We briefly summarize the session-typed process calculus \APCP (Asynchronous Priority-based Classical Processes).
We introduce its syntax, semantics, and type system.
\APCP supports tail-recursion, but since \LAST* does not we omit it.
We also summarize results: type preservation (\Cref{t:LASTn:APCP:tp}) and deadlock-freedom (\Cref{t:LASTn:APCP:df}).
\Cref{c:APCP} discusses all details of \APCP.

\paragraph{Syntax.}

\ExecuteMetaData[\fileLMCS]{LASTn:APCP:syntax}

\paragraph{Operational semantics.}

\ExecuteMetaData[\fileLMCS]{LASTn:APCP:semantics}

\paragraph{Type system.}

\ExecuteMetaData[\fileLMCS]{LASTn:APCP:types}

\paragraph{Results.}

\ExecuteMetaData[\fileLMCS]{LASTn:APCP:results}
\fi

\subsection{Towards a Faithful Translation of \protect\LAST* into APCP}
\label{s:LASTn:LAST:trans}

\ExecuteMetaData[\fileLMCS]{LASTn:LAST:trans}

\section{\texorpdfstring{\LASTn}{LASTn} and a Faithful Translation into \texorpdfstring{\APCP}{APCP}}
\label{s:LASTn:LASTn}

\ExecuteMetaData[\fileLMCS]{LASTn:LASTn}

\subsection{The Language of \texorpdfstring{\LASTn}{LASTn}}
\label{s:LASTn:LASTn:language}

\ExecuteMetaData[\fileLMCS]{LASTn:LASTn:language}

\subsection{Faithfully Translating \texorpdfstring{\LASTn}{LASTn} into \texorpdfstring{\APCP}{APCP}}
\label{s:LASTn:LASTn:trans}

\ExecuteMetaData[\fileLMCS]{LASTn:LASTn:trans}

\subsection{Deadlock-free \texorpdfstring{\LASTn}{LASTn}}
\label{s:LASTn:LASTn:df}

\ExecuteMetaData[\fileLMCS]{LASTn:LASTn:df}

\section{Related Work}
\label{s:LASTn:rw}

\ExecuteMetaData[\fileLMCS]{LASTn:rw}

\section{Closing Remarks}
\label{s:LASTn:concl}

\ExecuteMetaData[\fileLMCS]{LASTn:concl}

 \LASTnfalse
\alphalambdatrue 
\chapter{Bunched Functions as Processes}
\label{c:alphalambda}

This chapter discusses an interpretation of the logic of bunched implications (\BI) as a typed functional calculus.
It frames the calculus as a prototype programming language.
This way, an important litmus test is obtained for the interpretation of \BI as \pibi, a session-typed message-passing process calculus: a faithful translation would show that the process calculus can accurately model real programming languages and concepts.
Thus, this chapter addresses the following research question, introduced in \Cref{s:intro:contrib:lambda:alphalambda}:

\requAL*

This chapter is the result of a collaboration between myself and Dan Frumin, Emanuele D'Osualdo, and Jorge A.\ Pérez.
Here, I focus on the translation part of our publication~\cite{conf/oopsla/FruminDHP22} to which I contributed significantly.
\ifappendix
I include a self-contained summary of \BI and \pibi (\Cref{s:alphalambda:pibi}); omitted details and proofs can be found in the chapter dedicated to \pibi (\Cref{c:piBI}).
The publication contains additional material: a denotational semantics to derive behavioral equivalence and provenance tracking (omitted as a whole).
\else
This chapter relies on the definition and results of \BI and \pibi, discussed in detail in \Cref{c:piBI} on \Cpageref{c:piBI}.
Additional material (a denotation semantics to define behavioral equivalence and provenance tracking) can be found in~\cite{thesis/vdHeuvel24ex}.
\fi

\section{Introduction}
\label{s:alphalambda:intro}

The \picalc~\cite{book/Milner89,journal/ic/MilnerPW92} is a paradigmatic process calculus for studying message-passing concurrency, i.e., the interaction between concurrent processes collaborating by exchanging messages.
Studying message-passing concurrency as in the \picalc is important, as it reveals fine-grained and precise properties of systems that coordinate by exchanging messages.
However, due to this level of precision in the \picalc, it is not always obvious how results for the \picalc can be applied to real-world programming.

One way to gain insights from the \picalc is to design programming calculi that are more coarse and perhaps less expressive, but more intuitive.
Such calculi can showcase more directly the properties of programs gained from formal study.
Ideally, this idea can be strengthened by translation of the more intuitive calculus into the \picalc.
This then shows that the fine-grained features of the \picalc suitably handle the features of the coarser calculus, possibly enabling the transference of some of the properties of the \picalc.

This story also works in the opposite direction.
Suppose given a calculus that is intuitive to work with.
Because of the coarseness of its design, it may not be clear how it can be implemented as a programming language (e.g., how to design a compiler).
A translation into the \picalc can then reveal more precise and fine-grained principles, giving insights in how to achieve such an implementation.

\paragraph{Our contributions.}

In this chapter, we study such a translation.
In particular, we study the \alcalc~\cite{journal/bsl/OHearnP99,journal/jfp/OHearn03,book/Pym13}: a variant of the \lamcalc, derived from a Curry-Howard interpretation of the logic of bunched implications (\BI).
The \alcalc is interesting, because it inherits from \BI a fine-grained resource control based on ownership.
Though such properties are an interesting topic for study on their own, in this chapter the focus lies on translating the \alcalc into \piBI.
The calculus \piBI is a variant of the \picalc~\cite{conf/oopsla/FruminDHP22}, also derived from a Curry-Howard interpretation of \BI.
As such, \piBI is also interesting on its own, as it implements \BI's management of ownership explicity in the process language.
In contrast, the \alcalc handles ownership management more coarsely and implicitly.
Here, we show how \piBI can serve as a fine-grained explanation of the management of ownership in the \alcalc inherited from \BI.
Details on \piBI and the management of ownership derived from \BI can be found in \Cref{c:alphalambda}.

In \Cref{s:alphalambda:alphalambda} we introduce the \alcalc\ifappendix, and in \Cref{s:alphalambda:pibi} we briefly recall \piBI\fi.
In \Cref{s:alphalambda:trans} we define a translation from the \alcalc into \piBI, and show that this translation satisfies \emph{operational correspondences}: established correctness properties of translations that ensure that the behavior of translated programs precisely captures the behavior of source programs, attesting to the faithfulness of the translation.
\Cref{s:alphalambda:rw,s:alphalambda:concl} briefly discuss related work and conclusions, respectively.

\section{The \texorpdfstring{\alamcalc}{alpha-lambda-calculus}}
\label{s:alphalambda:alphalambda}

Here, we introduce the \alamcalc: syntax (\Cref{s:alphalambda:alphalambda:syntax}), semantics (\Cref{s:alphalambda:alphalambda:reduction}), and type system (\Cref{s:alphalambda:alphalambda:typeSys}), as introduced by O'Hearn and Pym~\cite{journal/bsl/OHearnP99,journal/jfp/OHearn03,book/Pym13}.

\subsection{Syntax}
\label{s:alphalambda:alphalambda:syntax}

\ExecuteMetaData[\fileOOPSLA{translation}]{alphalambda:alphalambda:syntax}

\subsection{Semantics}
\label{s:alphalambda:alphalambda:reduction}

\ExecuteMetaData[\fileOOPSLA{translation}]{alphalambda:alphalambda:reduction}

\subsection{Type System}
\label{s:alphalambda:alphalambda:typeSys}

The type system for the \alcalc is derived from the natural deduction presentation of \BI.
It follows the idea of the simply-typed \lamcalc, where a term requires variables of a certain type to provide a behavior of a certain type.
As such, \alcalc typing judgments are of the form $\Delta \vdash M : A$, where the term $M$ provides a behavior of type $A$ by using the variables in the typing context $\Delta$.

\begin{figure}[p]
    \def\MathparLineskip{\lineskip=3pt}
    \def\defaultHypSeparation{\hskip1ex}
    \ExecuteMetaData[\fileOOPSLA{calculus}]{alphalambda:alphalambda:types}

    \dashes

    Typing:
    \vspace{-1.5ex}
    \begin{mathpar}
        \begin{bussproof}[typ-id]
            \bussAx{
                x:A \vdash x : A
            }
        \end{bussproof}
        \and
        \begin{bussproof}[typ-cong]
            \bussAssume{
                \Delta \vdash M : A
            }
            \bussAssume{
                \Delta \equiv \Theta
            }
            \bussBin{
                \Theta \vdash M : A
            }
        \end{bussproof}
        \and
        \begin{bussproof}[typ-weaken]
            \bussAssume{
                \bunchCtx{\Gamma}[\Delta] \vdash M : A
            }
            \bussUn{
                \bunchCtx{\Gamma}[\Delta ; \Delta'] \vdash M : A
            }
        \end{bussproof}
        \and
        \begin{bussproof}[typ-contract]
            \bussAssume{
                \bunchCtx{\Gamma}[\idx{\Delta}{1} ; \idx{\Delta}{2}] \vdash M : A
            }
            \bussUn{
                \bunchCtx{\Gamma}[\Delta] \vdash M \{ x/\idx{x}{1} , x/\idx{x}{2} \mid x \in \fv(\Delta) \} : A
            }
        \end{bussproof}
        \and
        \begin{bussproof}[typ-wand-I]
            \bussAssume{
                \Delta , x:A \vdash M : B
            }
            \bussUn{
                \Delta \vdash \lam x . M : A \wand B
            }
        \end{bussproof}
        \and
        \begin{bussproof}[typ-impl-I]
            \bussAssume{
                \Delta ; x:A \vdash M : B
            }
            \bussUn{
                \Delta \vdash \alpha x . M : A \to B
            }
        \end{bussproof}
        \and
        \begin{bussproof}[typ-wand-E]
            \bussAssume{
                \Delta \vdash M : A \wand B
            }
            \bussAssume{
                \Theta \vdash N : A
            }
            \bussBin{
                \Delta , \Theta \vdash M\ N : B
            }
        \end{bussproof}
        \and
        \begin{bussproof}[typ-impl-E]
            \bussAssume{
                \Delta \vdash M : A \to B
            }
            \bussAssume{
                \Theta \vdash N : A
            }
            \bussBin{
                \Delta ; \Theta \vdash M \at N : B
            }
        \end{bussproof}
        \and
        \begin{bussproof}[typ-emp-I]
            \bussAx{
                \mEmpty \vdash \mUnit : \mOne
            }
        \end{bussproof}
        \and
        \begin{bussproof}[typ-true-I]
            \bussAx{
                \aEmpty \vdash \aUnit : \aOne
            }
        \end{bussproof}
        \and
        \begin{bussproof}[typ-emp-E]
            \bussAssume{
                \Delta \vdash M : \mOne
            }
            \bussAssume{
                \bunchCtx{\Gamma}[\mEmpty] \vdash N : A
            }
            \bussBin{
                \bunchCtx{\Gamma}[\Delta] \vdash \tLet \mUnit = M \tIn N : A
            }
        \end{bussproof}
        \and
        \begin{bussproof}[typ-true-E]
            \bussAssume{
                \Delta \vdash M : \aOne
            }
            \bussAssume{
                \bunchCtx{\Gamma}[\aEmpty] \vdash N : A
            }
            \bussBin{
                \bunchCtx{\Gamma}[\Delta] \vdash \tLet \aUnit = M \tIn N : A
            }
        \end{bussproof}
        \and
        \begin{bussproof}[typ-sep-I]
            \bussAssume{
                \Delta \vdash M : A
            }
            \bussAssume{
                \Theta \vdash N : B
            }
            \bussBin{
                \Delta , \Theta \vdash \<M,N\> : A \sep B
            }
        \end{bussproof}
        \and
        \begin{bussproof}[typ-conj-I]
            \bussAssume{
                \Delta \vdash M : A
            }
            \bussAssume{
                \Theta \vdash N : B
            }
            \bussBin{
                \Delta ; \Theta \vdash (M,N) : A \land B
            }
        \end{bussproof}
        \and
        \begin{bussproof}[typ-sep-E]
            \bussAssume{
                \Delta \vdash M : A \sep B
            }
            \bussAssume{
                \bunchCtx{\Gamma}[x:A , y:B] \vdash N : C
            }
            \bussBin{
                \bunchCtx{\Gamma}[\Delta] \vdash \tLet \<x,y\> = M \tIn N : C
            }
        \end{bussproof}
        \and
        \begin{bussproof}[typ-conj-E]
            \bussAssume{
                \Delta \vdash M : A_1 \land A_2
            }
            \bussAssume{
                i \in \{1,2\}
            }
            \bussBin{
                \Delta \vdash \pi_i M : A_i
            }
        \end{bussproof}
        \and
        \begin{bussproof}[typ-disj-I]
            \bussAssume{
                \Delta \vdash M : A_i
            }
            \bussAssume{
                i \in \{1,2\}
            }
            \bussBin{
                \Delta \vdash \tSel_i(M) \vdash A_1 \lor A_2
            }
        \end{bussproof}
        \and
        \begin{bussproof}[typ-disj-E]
            \bussAssume{
                \Delta \vdash M : A_1 \lor A_2
            }
            \bussAssume{
                \bunchCtx{\Gamma}[x_1:A_1] \vdash N_1 : C
            }
            \bussAssume{
                \bunchCtx{\Gamma}[x_2:A_2]
                \vdash N_2 : C
            }
            \bussTern{
                \bunchCtx{\Gamma}[\Delta] \vdash \tCase M \tOf \{ 1(x_1):N_1 , 2(x_2):N_2 \} : C
            }
        \end{bussproof}
        \and
        \begin{bussproof}[typ-cut]
            \bussAssume{
                \Delta \vdash M : A
            }
            \bussAssume{
                \bunchCtx{\Gamma}[x:A] \vdash N : C
            }
            \bussBin{
                \bunchCtx{\Gamma}[\Delta] \vdash N \{ M/x \} : C
            }
        \end{bussproof}
    \end{mathpar}

    \caption{Types and typing rules for the \alcalc.}
    \label{f:alphalambda:typing}
\end{figure}

The top of \Cref{f:alphalambda:typing} gives types, bunches, and bunched contexts; we explain the behavior associated with types when we discuss the typing rules below.%
\ExecuteMetaData[\fileOOPSLA{calculus}]{alphalambda:alphalambda:bunch}

\Cref{f:alphalambda:typing} also gives the typing rules for the \alamcalc.
Rule~\ruleLabel{typ-id} types a variable of arbitrary type.
Rule~\ruleLabel{typ-cong} closes typing under \ExecuteMetaData[\fileOOPSLA{calculus}]{alphalambda:alphalambda:bunchEquiv}
Rule~\ruleLabel{typ-weaken} discards resources from a bunch joined by `$\,;\,$' (i.e., additively combined).
Rule~\ruleLabel{typ-contract} duplicates resources under `$\,;\,$'; notation $\idx{\Delta}{1}$ denotes a copy of $\Delta$ with each variable annotated with `$\_^{(1)}$'.

Rule~\ruleLabel{typ-wand-I} (resp.~\ruleLabel{typ-impl-I}) types a multiplicative (resp.\ additive) abstraction.
Notice how the resources are joined with `$\,,\,$' (resp.~`$\,;\,$'), thus limiting (resp.\ enabling) the application of Rules~\ruleLabel{typ-weaken} and~\ruleLabel{typ-contract}; this principle holds for all other rules.
Rules~\ruleLabel{typ-wand-E} and~\ruleLabel{typ-impl-E} type function application.
Rules~\ruleLabel{typ-empty-I} and~\ruleLabel{typ-true-I} type the units, and Rule~\ruleLabel{typ-empty-E} and~\ruleLabel{typ-true-E} type unit-lets.
Rule~\ruleLabel{typ-sep-I} types multiplicative pairs, and Rule~\ruleLabel{typ-sep-E} types their unpacking.
Rule~\ruleLabel{typ-conj-I} types additive pairs, and Rule~\ruleLabel{typ-conj-E} types their projection.
Rule~\ruleLabel{typ-disj-I} types selection, and Rule~\ruleLabel{typ-disj-E} types branching.
Finally, Rule~\ruleLabel{typ-cut} types substitution; note that this rule is admissible (cf.\ \cite{journal/jfp/OHearn03}).

\ifappendix
\section{\texorpdfstring{\piBI}{piBI}: a Summary}
\label{s:alphalambda:pibi}

Similar to how the \alamcalc is derived from the natural deduction presentation of \BI, \piBI is a variant of the session-typed \picalc derived from the sequent calculus presentation of \BI.
In this section, we give a brief summary of \piBI; see \Cref{c:piBI} for details and proofs.

\ExecuteMetaData[\fileOOPSLA{calculus}]{alphalamda:piBI}

\paragraph{Syntax.}

\ExecuteMetaData[\fileOOPSLA{calculus}]{alphalamda:piBI:syntax}

\paragraph{Semantics.}

\ExecuteMetaData[\fileOOPSLA{calculus}]{alphalamda:piBI:semantics}

\paragraph{Typing.}

\ExecuteMetaData[\fileOOPSLA{calculus}]{alphalamda:piBI:typing}
\fi

\section{Translating the \texorpdfstring{\alamcalc}{alpha-lambda-caluclus} into \texorpdfstring{\piBI}{piBI}}
\label{s:alphalambda:trans}

\ExecuteMetaData[\fileOOPSLA{translation}]{alphalambda:trans}

\subsection{The Translation}
\label{s:alphalambda:trans:trans}

\ExecuteMetaData[\fileOOPSLA{translation}]{alphalambda:trans:trans}

\subsection{Operational Correspondence}
\label{s:alphalambda:trans:oc}

\ExecuteMetaData[\fileOOPSLA{translation}]{alphalambda:trans:oc}

\subsubsection{Completeness}
\label{s:alphalambda:trans:oc:completeness}

\ExecuteMetaData[\fileOOPSLA{translation}]{alphalambda:trans:oc:completeness}

\subsubsection{Soundness}
\label{s:alphalambda:trans:oc:soundness}

\ExecuteMetaData[\fileOOPSLA{translation}]{alphalambda:trans:oc:soundness}

\section{Related Work}
\label{s:alphalambda:rw}

\ExecuteMetaData[\fileOOPSLA{relwork}]{alphalambda:rw}

\section{Conclusions}
\label{s:alphalambda:concl}

We have presented a translation from O'Hearn and Pym's \alamcalc---a variant of the \lamcalc derived from the logic of bunched implications (\BI)---to our own \piBI---a variant of the \picalc derived from \BI.
As main result, we have shown that the translation is operationally correct: translated processes preserve the behavior of source terms (completeness) and no more (soundness).
Hence, the translation induces an adequate explanation of the implicit management of resource ownership in the \alamcalc in terms of explicit ownership management in \piBI.

 \alphalambdafalse

\addtocontents{toc}{\protect\hrulefill\protect\leavevmode\par}
\mypart[Multiparty Session Types: Distributed and Asynchronous]{Multiparty Session Types: Distributed and Asynchronous}
\label{p:mpst}

\relProjtrue 
\chapter{The Global and Local Perspective: Relative Types}
\label{c:relProj}

This chapter discusses multiparty session types (\MPSTs), from a global and local perspective.
That is, \MPSTs are usually expressed as global types that describe communication protocols from a vantage point.
But confirming protocol conformance is usually done on a protocol participant level, so a perspective of the protocol that is ``local'' to the participant is needed.
Prior methods of obtaining such local perspectives have been shown to be problematics, so this chapter explores new ways of obtaining local perspectives.
Hence, it answers the following research question, introduced in \Cref{s:intro:contrib:mpst:relProj}:

\requRelProj*

\section{Introduction}
\label{s:relProj:intro}

It is essential that the components of which modern distributed software systems consist communicate well.
This is easier said than done, as these components are often provided by several distributors, who prefer not to share the specifications of their software.
As a result, developers of distributed software are left to guess at the behavior of third-party components, as their public documentation is often incomplete or outdated.

It is thus of the utmost importance to find ways of taming the communication between the components of distributed systems, on which we rely in our everyday lives.
Designing methods for guaranteeing that components communicate correctly should be guided by two main principles: (1)~the techniques should be straightforward to implement in existing software, and (2)~verification techniques should be compositional.
Principle~(1) helps in convincing software distributors to invest in implementing the techniques, as the implementation requires little effort but gains greater accessibility making their software more attractive.
Principle~(2) means that guaranteeing the correctness of individual components implies the correctness of the whole system.
This entails that software distributors can guarantee that their software behaves correctly in any (correctly behaving) circumstances, without disclosing their software's design.
Moreover, it alleviates verifying correctness in entire distributed systems, which can be costly and complex, as distributed systems contain a lot of concurrency (i.e., things can happen in many orders).

A particularly salient approach to verifying the correctness of communication between distributed components can be found in multiparty session types (\MPSTs)~\cite{conf/popl/HondaYC08,journal/acm/HondaYC16}.
The theory of \MPSTs deals with communication \emph{protocols}, i.e., sequences of message exchanges between protocol \emph{participants}.
The messages may contain data, but also labels that determine which path to take when a protocol branches.
In \MPSTs, such protocols are often expressed as \emph{global types} that describe communication protocols between multiple participants from a vantage point.
The following global type expresses a protocol where a Server (`$s$') requests a Client (`$c$') to login through an Authorization Service (`$a$') (adapted from an example from~\cite{conf/popl/ScalasY19}):
\begin{align}
    \gtc{G_{\sff{auth}}} \deq \mu X \gtc. s{!}c \gtBraces{ \sff{login} \gtc. c{!}a (\msg \sff{pwd}<\sff{str}>) \gtc. a{!}s (\msg \sff{succ}<\sff{bool}>) \gtc. X \gtc, \sff{quit} \gtc. \tEnd }
    \label{eq:relProj:Gauth}
\end{align}
The global type $\gtc{G_{\sff{auth}}}$ denotes a recursive protocol (`$\mu X$'), meaning that it may repeat.
It starts with a message from $s$ to $c$ (`$s{!}c$') carrying a label \sff{login} or \sff{quit} that determines how the protocol proceeds.
In the \sff{login} case, the protocol proceeds with a message from $c$ to $a$ carrying a label \sff{pwd} and a value of type \sff{str} (`$\msg \sff{pwd}<\sff{str}>$').
This is followed by a message from $a$ to $s$ carrying the label \sff{succ} with a \sff{bool} value, after which the protocol repeats (`$X$').
In the \sff{quit} case, the protocol ends (`$\tEnd$').

The idea of verification with \MPSTs is that a protocol participant's implementation as a distributed component is checked for compliance to its part in a global type.
If all components in a system comply to their respective parts in a global type, compositionality of such \emph{local} verifications then guarantees that the whole system complies to the protocol.
As global types may contain information that is irrelevant to one particular participant, local verification can be simplified by only considering those interactions of the global type that are relevant to the participant under scrutiny, i.e., by taking a perspective that is \emph{local} to the participant.
Such local perspectives can also be helpful to guide the development of components.

\begin{figure}[t]
    \begin{minipage}[t]{0.45\linewidth}
        \centering
        \begin{tikzpicture}
            \node (G) {$\gtc{G}$};
            \node [below=7mm of G] (Lq) {$L_q$};
            \node [left=of Lq] (Lp) {$L_p$};
            \node [right=of Lq] (Lr) {$L_r$};
            \node [below=7mm of Lp] (P) {$P$};
            \node [below=7mm of Lq] (Q) {$Q$};
            \node [below=7mm of Lr] (R) {$R$};

            \draw[->, thick] (G.south) -- (Lp);
            \draw[->, thick] (G.south) -- (Lq.north);
            \draw[->, thick] (G.south) -- (Lr);

            \draw[dashed,->] (Lp.south) -- (P.north);
            \draw[dashed,->] (Lq.south) -- (Q.north);
            \draw[dashed,->] (Lr.south) -- (R.north);
        \end{tikzpicture}
    \end{minipage}%
    \hfill
    \vline
    \hfill
    \begin{minipage}[t]{0.45\linewidth}
        \centering
        \begin{tikzpicture}
            \node (G) {$\gtc{G}$};
            \node [below=7mm of G] (QR) {$\rtc{R_{q,r}}$};
            \node [left=of QR] (PQ) {$\rtc{R_{p,q}}$};
            \node [right=of QR] (RP) {$\rtc{R_{r,p}}$};
            \node [below=7mm of PQ] (P) {$P$};
            \node [below=7mm of QR] (Q) {$Q$};
            \node [below=7mm of RP] (R) {$R$};

            \draw[->, thick] (G.south) -- (PQ);
            \draw[->, thick] (G.south) -- (QR);
            \draw[->, thick] (G.south) -- (RP);

            \draw[dashed,->] (PQ.south) -- (Q);
            \draw[dashed,->] (QR.south) -- (R);
            \draw[dashed,->] (RP.south) -- (P);
            \draw[dashed,->] (PQ.south) -- (P.north);
            \draw[dashed,->] (QR.south) -- (Q.north);
            \draw[dashed,->] (RP.south) -- (R.north);
        \end{tikzpicture}
    \end{minipage}
    \caption{Local projection (left) and relative projection (right) of global type $\gtc{G}$ onto its participants $p,q,r$, implemented by processes $P,Q,R$, respectively.}
    \ifrelProj \label{f:relProj:diag}
    \else \label{f:intro:relProj:diag} \fi
\end{figure}

The widely accepted method of obtaining local perspectives from global types is \emph{local projection}, which projects a global type onto a single participant, leading to a \emph{local type} that expresses only those message exchanges in the global type in which the participant is involved~\cite{conf/popl/HondaYC08,journal/acm/HondaYC16,conf/fossacs/YoshidaDBH10}.
\Cref{f:relProj:diag} (left) illustrates how this method is used in the local verification of processes implementing the roles of the participants of a global type.
However, as we discuss next, global types may contain communication patterns that cannot be expressed as local types.
Therefore, verification techniques that rely on local projection can only handle a limited class of \emph{well-formed} global types.

The class of global types that are well-formed for local projection excludes many nonsensical communication patterns, but it also excludes many useful patterns such as those occurring in $\gtc{G_{\sff{auth}}}$~\eqref{eq:relProj:Gauth} (as observed by Scalas and Yoshida~\cite{conf/popl/ScalasY19}).
This is because local projection cannot adequately handle \emph{non-local choices}.
In $\gtc{G_{\sff{auth}}}$, the exchange between $s$ and $c$ denotes a non-local choice for the protocol local to $a$: the global type states that $a$ should communicate with $c$ and $s$ only after $s$ sends to $c$ the label \sff{login}.

It is unfortunate that local projection cannot handle global types with non-local choices as in $\gtc{G_{\sff{auth}}}$~\eqref{eq:relProj:Gauth}: clearly, this global type expresses a protocol that is very useful in practice.
Hence, it is important to reconsider using local projection for compositional verification techniques.
This chapter introduces \emph{relative projection}, a new approach to obtaining local perspectives of global types.
Relative projection projects global types onto \emph{pairs of participants}, leading to a form of binary session types called \emph{relative types}.
This new projection enables the verification of a new class of \emph{relative well-formed} global types that includes global types with non-local choices such as $\gtc{G_{\sff{auth}}}$.
It does so by incorporating explicit coordination messages called \emph{dependencies} that force participants to synchronize on non-local choices.

\Cref{s:relProj:global,s:relProj:local} introduce global and relative types, respectively.
\Cref{s:relProj:relProj} introduces relative projection and the corresponding class of relative well-formed global types.
\Cref{s:relProj:wf} compares classical well-formedness based on local projection and our new relative well-formedness.
\Cref{s:relProj:rw,s:relProj:concl} discuss related work and conclusions, respectively.

\section{The Global Perspective}
\label{s:relProj:global}

\ExecuteMetaData[\fileSCICO]{relProj:global}

\section{The Local Perspective}
\label{s:relProj:local}

\ExecuteMetaData[\fileSCICO]{relProj:local}

\section{Relative Projection and Well-formedness}
\label{s:relProj:relProj}

\ExecuteMetaData[\fileSCICO]{relProj:relProj}

\section{Comparing Well-formedness}
\label{s:relProj:wf}

\ExecuteMetaData[\fileSCICO]{relProj:comp}

\section{Related Work}
\label{s:relProj:rw}

\ExecuteMetaData[\fileSCICO]{relProj:rw1}

\ExecuteMetaData[\fileSCICO]{relProj:rw2}

\section{Conclusions}
\label{s:relProj:concl}

We have introduced a new method to obtain local perspectives from global types.
This methods uses \emph{relative projection} to project a global type onto \emph{pairs} of its participants, yielding a form of binary session type dubbed \emph{relative types}.

Relative types and projection support an expressive class of global types (those that are \emph{relative well-formed}), because of the way they deal with the \emph{non-local choices} that influence the interactions between pairs of participants: they make these decision points that are implicit in the global type explicit by means of coordination messages called \emph{dependencies}.

Our new relative projection improves over traditional methods of \emph{local projection}.
First, relative projection yields relative types that are more directly compatible with binary session type theories.
Second, the class of relative well-formed global types is arguably more expressive than the the class of global types that are well-formed for local projection.
Though these classes are incomparable, we have demonstrated with several examples that many global types that are not relative well-typed can be transformed to equivalent, relative well-formed global types.
In future work, we plan to improve relative projection to subsume merge-based approaches entirely.

The methods introduced in this chapter are readily usable for the verification of implementations of multiparty session types, be it static (\Cref{c:mpstAPCP}) or dynamic (\Cref{c:mpstMon}).

 \relProjfalse
\mpstAPCPtrue 
\chapter{Binary Session Types for Distributed Multiparty Session}
\label{c:mpstAPCP}

This chapter addresses the verification of distributed systems, in particular as implementations of multiparty session types (\MPSTs).
The chapter then introduces a framework in which such distributed models of \MPSTs can be analyzed.
Instead of re-inventing the wheel, the framework relies on message-passing processes in \APCP and their derived correctness properties.
It thus answers the following research question, introduced in \Cref{s:intro:contrib:mpst:mpstAPCP}:

\requMpstAPCP*

\ifappendix
This chapter includes a self-contained summary of \MPSTs (\Cref{s:mpstAPCP:mpst}) and of \APCP (\Cref{s:mpstAPCP:APCP}).
Omitted details and proofs can be found in the chapters dedicated to these subjects (\Cref{c:relProj,c:APCP}, respectively).
\else
This chapter relies on the definition and results of \APCP, presented in detail in \Cref{c:APCP}.
\fi
The chapter also relies on global types and relative projection, presented in detail in \Cref{c:relProj}; a summary is included in \Cref{s:mpstAPCP:mpst} with minor adjustments compared to \Cref{c:relProj}.

\section{Introduction}
\label{s:mpstAPCP:intro}

\ExecuteMetaData[\fileSCICO]{mpstAPCP:intro}

\ifappendix
\section{\texorpdfstring{\APCP}{APCP}: a Summary}
\label{s:mpstAPCP:APCP}

In this section we briefly summarize the session-typed process calculus \APCP (Asynchronous Priority-based Classical Processes).
We introduce its syntax, semantics, and type system.
We also summarize results: type preservation (\Cref{t:mpstAPCP:APCP:tp}) and deadlock-freedom (\Cref{t:mpstAPCP:APCP:df}).
\Cref{c:APCP} discusses all details of \APCP.

\paragraph{Syntax.}

\ExecuteMetaData[\fileLMCS]{mpstAPCP:APCP:syntax}

\paragraph{Operational semantics.}

\ExecuteMetaData[\fileLMCS]{mpstAPCP:APCP:semantics}

\paragraph{Type system.}

\ExecuteMetaData[\fileLMCS]{mpstAPCP:APCP:types}

\paragraph{Results.}

\ExecuteMetaData[\fileLMCS]{mpstAPCP:APCP:results}

\paragraph{Reactivity.}

\ExecuteMetaData[\fileLMCS]{mpstAPCP:APCP:react}
\fi

\ifappendix
\paragraph{Example.}
\else
\section{The Authorization Protocol in \texorpdfstring{\APCP}{APCP}}
\label{s:mpstAPCP:APCPExample}
\fi

\ExecuteMetaData[\fileSCICO]{mpstAPCP:APCP:example}

\section{Global Types and Relative Types: a Summary}
\label{s:mpstAPCP:mpst}

In this section we briefly summarize \MPSTs global types and their projection onto relative types.
This corresponds to their detailed introduction in \Cref{c:relProj} with minor adjustments to accommodate routers.

\paragraph{Global types.}

\ExecuteMetaData[\fileSCICO]{mpstAPCP:mpst:global}

\paragraph{Relative types.}

\ExecuteMetaData[\fileSCICO]{mpstAPCP:mpst:relative}

\paragraph{Relative projection and well-formedness.}

\ExecuteMetaData[\fileSCICO]{mpstAPCP:mpst:relProj}

\section{Analyzing Global Types using Routers}
\label{s:mpstAPCP:analysis}

\ExecuteMetaData[\fileSCICO]{mpstAPCP:analysis}

\section{Routers in Action}
\label{s:mpstAPCP:examples}

\ExecuteMetaData[\fileSCICO]{mpstAPCP:examples}

\section{Related Work}
\label{s:mpstAPCP:rw}

\ExecuteMetaData[\fileSCICO]{mpstAPCP:rw}

\section{Conclusions}
\label{s:mpstAPCP:concl}

\ExecuteMetaData[\fileSCICO]{mpstAPCP:concl}

 \mpstAPCPfalse
\mpstMontrue \chapter[Monitors for Blackbox Implementations of Multiparty Session Types]{Monitors \texorpdfstring{\\}{} for Blackbox Implementations \texorpdfstring{\\}{} of Multiparty Session Types}
\label{c:mpstMon}

This chapter explores the verification of the correctness of distributed systems without access to the specification of system components.
That is, these systems consist of ``blackboxes''.
The idea is to confirm protocol conformance for multiparty session types (\MPSTs), where blackboxes implement the roles of protocol participants.
By equipping blackboxes with monitors that observe their behavior, this chapter studies the dynamic verification of distributed implementations of \MPSTs.
That is, the chapter addresses the following research question, introduced in \Cref{s:intro:contrib:mpst:mpstMon}:

\requMpstMon*

\ifappendix
This chapter includes a self-contained summary of \MPSTs (\Cref{s:mpstMon:synth}).
Omitted details can be found in the chapter dedicated to \MPSTs (\Cref{c:relProj}).
\else
This chapter relies on global types and relative projection (cf.\ \Cref{c:relProj} on \Cpageref{c:relProj}).
\fi

\section{Introduction}
\label{s:mpstMon:intro}

\ExecuteMetaData[\fileRV]{mpstMon:intro}

\section{Networks of Monitored Blackboxes}
\label{s:mpstMon:networks}

\ExecuteMetaData[\fileRV]{mpstMon:networks}

\section[Monitors for Blackboxes Synthesized from Global Types]{Monitors for Blackboxes Synthesized \texorpdfstring{\\}{} from Global Types}
\label{s:mpstMon:synth}

\ExecuteMetaData[\fileRV]{mpstMon:synth}

\section{Properties of Correct Monitored Blackboxes}
\label{s:mpstMon:properties}

\ExecuteMetaData[\fileRV]{mpstMon:properties}

\subsection{Satisfaction}
\label{s:mpstMon:properties:satisfaction}

\ExecuteMetaData[\fileRV]{mpstMon:properties:satisfaction}

\subsection{Soundness}
\label{s:mpstMon:properties:soundness}

\ExecuteMetaData[\fileRV]{mpstMon:properties:soundness}

\subsection{Transparency}
\label{s:mpstMon:properties:transparency}

\ExecuteMetaData[\fileRV]{mpstMon:properties:transparency}

\section{Conclusion}
\label{s:mpstMon:concl}

\ExecuteMetaData[\fileRV]{mpstMon:concl}

 \mpstMonfalse

\addtocontents{toc}{\protect\hrulefill\protect\leavevmode\par}
\chapter{Conclusions}
\label{c:conclusions}

Since each chapter includes conclusions of their own, I will only briefly conclude my dissertation as a whole.
In this thesis, I have addressed from different angles the following encompassing research question:

\theEncrequ*

\noindent
Let me rephrase this, not in terms of pushing boundaries, but in terms of using logical foundations as a starting point.
Linear logic identifies precisely an area of session types where processes/programs are very well-behaved: they satisfy important correctness criteria of which deadlock-freedom is most significant.
Starting in this area, we can change parts of the definitions of session types and processes/programs, such that we relax some restrictions or change the behavior of processes/programs.
The question is then: how and where can we do this without sacrificing correctness?

\Cref{c:APCP,c:clpi} take this approach most directly: the former switches from synchronous to asynchronous communication, and the latter reintroduces non-deterministic choice; in both cases, we prove that correctness properties are preserved.
\Cref{c:piBI} takes a different angle, by studying the area of session types identified by the logic of bunched implications and the correctness properties that hold therein.

\Cref{c:LASTn,c:alphalambda} connect the worlds of processes in \Cref{c:APCP,c:piBI} to the world of functional programs.
In particular, the former shows how guarantees gained from linear logic in the world of processes can be transferred to the world of functional programs, whereas the latter shows that the worlds of processes and of functional programs---both derived separately from the logic of bunched implications---are tightly connected.

\Cref{c:relProj,c:mpstAPCP} show how useful the insights gained in \Cref{c:APCP} can be, as they connect binary session types---connected to linear logic directly---to multiparty session types---further from logic, but more useful in practice.
\Cref{c:mpstMon} then shows that this connection not only applies to the static verification of multiparty session types implemented as processes, but also to dynamic verification.

All in all, this dissertation demonstrates that the established logical foundations in linear logic are important and useful, but that they by no means should be viewed as absolute ``laws of message-passing''.
Session type systems derived from logic form a solid starting point in developing systems that address more specific features (such as asynchrony and non-determinism).
As such, logic contributes greatly to providing correctness properties.
Moreover, it provides a unified language design on which to build, perhaps paving the way for an idealistic framework in which many approaches can be united to address a multitude of important message-passing aspects at once.

\thumbfalse

\addtocontents{toc}{\protect\hrulefill\protect\leavevmode\par}
\printbibliography[heading=bibintoc,label=refs]

\addcontentsline{toc}{chapter}{Glossary}
\printglossary[title={Glossary},nogroupskip,style=mylongstyle]

\addtocontents{toc}{\protect\hrulefill\protect\leavevmode\par}
\chapter*{Acknowledgments}
\markboth{Acknowledgments}{}
\addcontentsline{toc}{chapter}{Acknowledgments}

As far as I know, there is no way to write ``concurrently'', so unfortunately I need to ``linearize'' my unordered acknowledgments.
I shall switch between English and Dutch, reflecting the language I spreak with people.
Groups of people are listed alphabetically.

My PhD supervisor, Jorge Pérez, deserves a lot of gratitude.
Starting with my first email asking for advice on getting a PhD position in TCS, he has always supported me and my career with 100\% conviction.
Through the years, we found a good way of collaborating, even though I worked from Amsterdam most of the time.
Jorge's abilities and insights in the academic world are inspiring and humbling.
Jorge, take a day or two off every now and then ;).
I should also express my gratitude to my second supervisor, Gerard Renardel de Lavalette; we didn't spend much time together, but I did enjoy it when we did.

I would also like to express my deepest gratitude to the reading committee of my dissertation: Marieke Huisman, Alexander Lazovik, Nobuko Yoshida.
I know the thesis got rather lengthy, so I greatly appreciate the time they took from their busy schedules to read it.

Zonder mijn ouders, Carolien en Teije, zou ik hier niet zijn gekomen (letterlijk).
Mijn hele leven hebben ze mij zonder (grote) vraagtekens altijd gesteund in wat ik ook doe.
Pap en mam, hoewel wij als een goed Nederlands gezin elkaar soms maanden niet zien, denk ik altijd aan jullie en ben ik jullie altijd dankbaar voor hoe jullie onvoorwaardelijk voor me klaarstaan; ik hou van jullie!
Ik hoop dat mijn vader ermee kan leven dat zijn investeringen in mij niet veel zullen bijdragen aan zijn pensioen als ik op het academische pad blijf \textellipsis

Arianne is de liefde van m'n leven, en mijn steun en toeverlaat.
Hoe ze het uithoudt met mij is me een compleet raadsel, zeker gezien mijn carrièrekeuze niet de meest optimale leefsituaties oplevert voor haar.
Arianne, je bent een enorme steun geweest tijdens de zwaarste perioden van mijn PhD, en ik hoop dat ik je niet te veel heb belast met mijn gestress.
Laten we nog heel veel beleven samen, ik hou van je!

Tom (or, as he prefers to be referred to by others than me, Tomislav) is one of my best friends.
We met each other at the Master of Logic in Amsterdam, and Tom quickly became an excellent friend at a time where I really needed one.
Tom, I am very grateful to have you in my life, and although you are moving back to Zagreb, I don't think you will ever get rid of me.
Not only has Tom submerged me in Croatian culture at his wedding, he has also introduced me to his many friends from the Balkans who live in Amsterdam; I consider them my own friends and fondly refer them as the ``Balkan Boys'': Frane, Granit, Milan, Nino, and Vojin.

Met acht jaar jonger, gingen ik en mijn broertje, Sam, de langste tijd niet echt goed met elkaar om.
Tijdens mijn PhD, echter, zijn we meer en meer gaan hangen, en nu beschouw ik Sam als een van mijn beste vrienden.
Wat we ook doen samen---gamen, deels mislukte fiets-/kampeervakanties, helpen met verhuizen---, het geeft me altijd veel plezier.
Sam, ik ben heel blij dat we zo goed met elkaar omgaan de laatste tijd.
Natuurlijk ben ik mijn zus, Eva, ook dankbaar.
Ik heb geleerd politieke discussies met haar uit de weg gegaan (vooral omdat ze nog progressiever is dan ik \textellipsis), maar het is altijd fijn samen; dankzij haar ken ik ook haar partner, Anne, die me ervan heeft overtuigd te gaan GM'en voor RPGs, en natuurlijk heb ik mijn twee fantastische neefjes, Siep en Willem, aan hen te danken.

Tijmen is de beste vriend die ik aan mijn tijd bij de bachelorstudieverening via over heb gehouden.
Toendertijd hebben we veel gehangen, waarbij we altijd muziek maakten, luisterden, en ontdekten.
Een tijdje hebben we weinig contact gehad, maar ik ben erg blij dat we elkaar hervonden hebben door onze stoner-band-projecten met Arianne en Niels.
Ik ben ook heel blij met onze recente regelmatige game-/hang-sessies.
Niet te vergeten: Tijmen, enorm bedankt voor het ontwerpen van de kaft van mijn thesis, je hebt het fantastisch gedaan!

Ik moet ook de andere fantastische vrienden noemen die ik aan mijn via-tijd over heb gehouden, en die ik nog steeds met enige regelmaat zie: Bram, Chiel, Elise, Fabiën, Iris, Jelte, Jorn, Micha, Robbert, Sander.
Het is altijd een plezier om met jullie wat te drinken, bordspelletjes of RPGs te spelen, op city-trips te gaan, etcetera.

Of course, I want to thank my colleagues at the University of Groningen: Alen, Anton, Dan, Helle, Joe, Juan, Revantha, Tina, Vitor.
During my PhD I also got to know researchers outside of Groningen, both through visits and through fruitful collaboration: Daniele, Emanuele, Farzaneh, Stephanie; thank you, and I hope we get to work together more in the future.

To anyone I didn't mention explicitly, family, friends, and colleagues: thank you!

\ifappendix

\thumbtrue

\appendix
\addtocontents{toc}{\protect\hrulefill\protect\leavevmode\par}
\part{Appendices}

\chapter{Proofs for \crtCref{c:basepi}}
\label{ac:basepi}

This appendix details proofs of the results in \Cref{c:basepi}.

\tPiCalcSubjCong*

\begin{proof}
    By induction on the derivation of $P \equiv Q$.
    The inductive cases, which apply congruence under contexts, follow from the IH straightforwardly.
    I discuss each base case, corresponding to the axioms in \Cref{d:basepi:strcong}, separately:
    \begin{itemize}

        \item
            \textbf{(Rule~\ruleLabel{cong-alpha})}
            Then $P \equiv_\alpha Q$, i.e., $P$ and $Q$ are equal up to renaming of bound names.
            Hence, reflect this renaming in the derivation of $\vdash P \typInf \Delta$ to obtain a derivation of $\vdash Q \typInf \Delta$.

        \item
            \textbf{(Rule~\ruleLabel{cong-par-unit})}
            By inversion of Rules~\ruleLabel{typ-par} and~\ruleLabel{typ-inact}, $\vdash P \typInf \Delta$, $\vdash \0 \typInf \emptyset$, and $\vdash P \| \0 \typInf \Delta$, immediately showing the thesis.

        \item
            \textbf{(Rule~\ruleLabel{cong-par-comm})}
            By inversion of Rule~\ruleLabel{typ-par}, $\vdash P \typInf \Delta_1$, $\vdash Q \typInf \Delta_2$, and $\vdash P \| Q \typInf \Delta_1 , \Delta_2$.
            By Rule~\ruleLabel{typ-par}, $\vdash Q \| P \typInf \Delta_1 , \Delta_2$, showing the thesis.

        \item
            \textbf{(Rule~\ruleLabel{cong-par-assoc})}
            By inversion of two applications of Rule~\ruleLabel{typ-par}, $\vdash P \typInf \Delta_1$, $\vdash Q \typInf \Delta_2$, $\vdash R \typInf \Delta_3$, and $\vdash P \| (Q \| R) \typInf \Delta_1 , \Delta_2 , \Delta_3$.
            By two applications of Rule~\ruleLabel{typ-par}, $\vdash (P \| Q) \| R \typInf \Delta_1 , \Delta_2 , \Delta_3$, showing the thesis.

        \item
            \textbf{(Rule~\ruleLabel{cong-scope})}
            By inversion of Rules~\ruleLabel{typ-par} and~\ruleLabel{typ-res}, $\vdash P \typInf \Delta_1$, \mbox{$\vdash Q \typInf \Delta_2 , x:A , y:\ol{A}$}, $\vdash \pRes{xy} Q \typInf \Delta_2$, and $\vdash P \| \pRes{xy} Q \typInf \Delta_1 , \Delta_2$.
            Since $x,y \notin \fn(P)$, \mbox{$x,y \notin \dom(P)$}.
            By Rules~\ruleLabel{typ-par} and~\ruleLabel{typ-res}, $\vdash \pRes{xy} ( P \| Q ) \typInf \Delta_1 , \Delta_2$, showing the thesis.

        \item
            \textbf{(Rule~\ruleLabel{cong-res-comm})}
            By inversion of two applications of Rule~\ruleLabel{typ-res}, \mbox{$\vdash P \typInf \Delta , x:A , y:\ol{A} , z:B , w:\ol{B}$}, and $\vdash \pRes{xy} \pRes{zw} P \typInf \Delta$.
            By two applications of Rule~\ruleLabel{typ-res}, $\vdash \pRes{zw} \pRes{xy} P \typInf \Delta$, showing the thesis.

        \item
            \textbf{(Rule~\ruleLabel{cong-res-symm})}
            By inversion of Rule~\ruleLabel{typ-res}, $\vdash P \typInf \Delta , x:A , y:\ol{A}$, and $\vdash \pRes{xy} P \typInf \Delta$.
            By involution of duality and Rule~\ruleLabel{typ-res}, $\vdash \pRes{yx} P \typInf \Delta$, showing the thesis.

        \item
            \textbf{(Rule~\ruleLabel{cong-fwd-symm})}
            By inversion of Rule~\ruleLabel{typ-fwd}, $\vdash \pFwd [x<>y] \typInf x:A , y:\ol{A}$.
            By involution of duality and Rule~\ruleLabel{typ-fwd}, $\vdash \pFwd [y<>x] \typInf x:A , y:\ol{A}$, showing the thesis.

        \item
            \textbf{(Rule~\ruleLabel{cong-res-fwd})}
            By inversion of Rules~\ruleLabel{typ-res} and~\ruleLabel{typ-fwd}, \mbox{$\vdash \pFwd [x<>y] \typInf x:A , y:\ol{A}$}, and $\vdash \pRes{xy} \pFwd [x<>y] \typInf \emptyset$.
            By Rule~\ruleLabel{typ-inact}, $\vdash \0 \typInf \emptyset$, showing the thesis.
            \qedhere

    \end{itemize}
\end{proof}

\tPiCalcSubjRedd*

\begin{proof}
    By induction on the derivation of $P \redd Q$.
    I discuss each case, corresponding to the rules in \Cref{d:basepi:redd}, separately:
    \begin{itemize}

        \item
            \textbf{(Rule~\ruleLabel{red-close-wait})}
            By inversion of Rules~\ruleLabel{typ-res}, \ruleLabel{typ-par}, \ruleLabel{typ-close}, and \ruleLabel{typ-wait}, $\vdash \pClose x[] \typInf x:\1$, $\vdash Q \typInf \Delta$ , $\vdash \pWait y() ; Q \typInf \Delta, y:\bot$, and $\vdash \pRes{xy} ( \pClose x[] \| \pWait y() ; Q ) \typInf \Delta$, immediately showing the thesis.

        \item
            \textbf{(Rule~\ruleLabel{red-send-recv})}
            By inversion of Rules~\ruleLabel{typ-res}, \ruleLabel{typ-par}, \ruleLabel{typ-send}, and~\ruleLabel{typ-recv}, $\vdash P \typInf \Delta_1$, $\vdash \pOut x[a,b] ; P \typInf \Delta_1 , x:A \tensor B , a:\ol{A} , b:\ol{B}$, $\vdash Q \typInf \Delta_2 , z:\ol{A} , y':\ol{B}$, $\vdash \pIn y(z,y') ; Q \typInf \Delta_2 , y:\ol{A} \parr \ol{B}$, and $\vdash \pRes{xy} ( \pOut x[a,b] ; P \| \pIn y(z,y') ; Q ) \typInf \Delta_1 , \Delta_2 , a:\ol{A} , b:\ol{B}$.
            It is straightforward to derive $\vdash Q \{a/z,b/y'\} \typInf \Delta_2 , a:\ol{A} , b:\ol{B}$.
            By Rule~\ruleLabel{typ-par}, $\vdash P \| Q \{a/z,b/y'\} \typInf \Delta_1 , \Delta_2 , a:\ol{A} , b:\ol{B}$, showing the thesis.

        \item
            \textbf{(Rule~\ruleLabel{red-sel-bra})}
            By inversion of Rules~\ruleLabel{typ-res}, \ruleLabel{typ-par}, \ruleLabel{typ-sel}, and~\ruleLabel{typ-bra}, $\vdash P \typInf \Delta_1$, $\vdash \pSel x[b] < j ; P \typInf \Delta_1 , x:\oplus_{i \in I} A_i , b:\ol{A_j}$, $\vdash Q_i \typInf \Delta_2 , y':\ol{A_i}$ for every~\mbox{$i \in I$}, \mbox{$\vdash \pBra y(y') > {\{ i . Q_i \}_{i \in I}} \typInf \Delta_2 , y:\&_{i \in I} \ol{A_i}$}, and \mbox{$\vdash \pRes{xy} ( \pSel x[b] < j ; P \| \pBra y(y') > {\{ i . Q_i \}_{i \in I}} ) \typInf \Delta_1 , \Delta_2 , b:\ol{A_j}$}.
            It is straightforward to derive $\vdash Q_j \{b/y'\} \typInf \Delta_2 , b:\ol{A_j}$.
            By Rule~\ruleLabel{typ-par}, \mbox{$\vdash P \| Q_j \{b/y'\} \typInf \Delta_1 , \Delta_2 , b:\ol{A_j}$}, showing the thesis.

        \item
            \textbf{(Rule~\ruleLabel{red-fwd})}
            By inverison of Rules~\ruleLabel{typ-res}, \ruleLabel{typ-par}, and \ruleLabel{typ-fwd}, \mbox{$\vdash \pFwd [x<>z] \typInf x:A , z:\ol{A}$}, $P \vdash \Delta , y:\ol{A}$, and $\vdash \pRes{xy} ( \pFwd [x<>z] \| P ) \typInf \Delta, z:\ol{A}$.
            It is straightforward to derive $\vdash P \{z/y\} \typInf \Delta , y:\ol{A}$, showing the thesis.

        \item
            \textbf{(Rule~\ruleLabel{red-cong})}
            By \Cref{t:basepi:subjCong}, $\vdash P' \typInf \Delta$.
            By the IH, $\vdash Q' \typInf \Delta$.
            By \Cref{t:basepi:subjCong}, $\vdash Q \typInf \Delta$, showing the thesis.

        \item
            \textbf{(Rule~\ruleLabel{red-res})}
            By inversion of Rule~\ruleLabel{typ-res}, $\vdash P \typInf \Delta , x:A , y:\ol{A}$, and $\vdash \pRes{xy} P \typInf \Delta$.
            By the IH, $\vdash Q \typInf \Delta , x:A , y:\ol{A}$.
            By Rule~\ruleLabel{typ-res}, $\vdash \pRes{xy} Q \typInf \Delta$, showing the thesis.

        \item
            \textbf{(Rule~\ruleLabel{red-par})}
            By inversion of Rule~\ruleLabel{typ-par}, $\vdash P \typInf \Delta_1$, $\vdash R \typInf \Delta_2$, and $\vdash P \| R \typInf \Delta_1 , \Delta_2$.
            By the IH, $\vdash Q \typInf \Delta_1$.
            By Rule~\ruleLabel{typ-par}, $\vdash Q \| R \typInf \Delta_1 , \Delta_2$, showing the thesis.
            \qedhere

    \end{itemize}
\end{proof}


\APCPtrue \chapter{Proofs for \crtCref{c:APCP}}
\label{ac:APCP}

This appendix details proofs of the results in \Cref{c:APCP}.

\section{Subject Reduction}

\tAPCPSC*

\ExecuteMetaData[\fileLMCS]{APCP:sc:proof}

\section{Deadlock-freedom}

\lAPCPUnfold*

\ExecuteMetaData[\fileLMCS]{APCP:unfold:proof}

\tAPCPProgress*

\ExecuteMetaData[\fileLMCS]{APCP:progress:proof}

\lAPCPNotLiveNoAction*

\ExecuteMetaData[\fileLMCS]{APCP:notLiveNoAction:proof}

\tAPCPDF*

\ExecuteMetaData[\fileLMCS]{APCP:df:proof}

\section{Reactivity}

\pLredd*

\ExecuteMetaData[\fileLMCS]{APCP:lredd:proof}

\tAPCPReact*

\ExecuteMetaData[\fileLMCS]{APCP:react:proof}

 \APCPfalse
\chapter{Appendices for \crtCref{c:clpi}}
\label{ac:clpi}

\section{Proof of Type Preservation}
\label{as:clpi:tp}

Here we prove type preservation for both semantics.
Type preservation consists of \emph{subject congruence} (structural congruence preserves typing) and \emph{subject reduction} (reduction preserves typing).
The former is ubiquitous for both semantics, while the latter requires separate proofs.

\subsection{Subjection Congruence}
\label{as:clpi:tp:sc}

\ExecuteMetaData[\fileAPLASAppPi]{clpi:tp:sc}

\subsection{Subject Reduction: Eager Semantics}
\label{as:clpi:tp:srEager}

The proofs relies on the follows auxiliary results:
\begin{itemize}

    \item
        \Cref{l:clpi:ctxType} infers the type of a name from how it appear in a process.

    \item
        \Cref{l:clpi:reddEager} infers typing after reduction and collapsing ND-contexts.

    \item
        \Cref{t:clpi:srEager} proves subject reduction.

\end{itemize}

\ExecuteMetaData[\fileAPLASAppPi]{clpi:tp:srEager}

\subsection{Subject Reduction: Lazy Semantics}
\label{as:clpi:tp:srLazy}

\ExecuteMetaData[\fileAPLASAppPi]{clpi:tp:srLazy}

\section{Proof of Deadlock-freedom}
\label{as:clpi:df}

\subsection{Eager Semantics}
\label{as:clpi:df:eager}

\ExecuteMetaData[\fileLICSAppPi]{clpi:dfEager}

\subsection{Lazy Semantics}
\label{as:clpi:df:lazy}

\ExecuteMetaData[\fileAPLASAppPi]{clpi:dfLazy}

\section{Proof of Bisimilarity Result}
\label{as:clpi:bisim}

\ExecuteMetaData[\fileAPLASAppBisim]{clpi:bisim}

\piBItrue \chapter{Proofs for \crtCref{c:piBI}}
\label{ac:piBI}

This appendix details the proofs of the results in \Cref{c:piBI}.%
\ExecuteMetaData[\fileOOPSLA{appendix/app-meta}]{piBI:proofs}

\section{Auxiliary Results}
\label{as:piBI:aux}

\ExecuteMetaData[\fileOOPSLA{appendix/app-meta}]{piBI:proofs:aux}

\section{Type Preservation}
\label{as:piBI:tp}

\ExecuteMetaData[\fileOOPSLA{appendix/app-meta}]{piBI:proofs:tp}

\section{Deadlock-freedom}
\label{as:piBI:df}

\ExecuteMetaData[\fileOOPSLA{appendix/app-meta}]{piBI:proofs:df}

\section{Weak Normalization}
\label{as:piBI:wn}

\ExecuteMetaData[\fileOOPSLA{appendix/app-meta}]{piBI:proofs:wn}

 \piBIfalse

\LASTntrue 
\chapter{Appendices for \crtCref{c:LASTn}}
\label{ac:LASTn}

This appendix details proofs of the results in \Cref{c:LASTn}.

\section[Self-contained Definition of \protect\LASTn and its Type System]{Self-contained Definition of \protect\LASTn \texorpdfstring{\\}{} and its Type System}
\label{as:LASTn:lang}

\ExecuteMetaData[\fileLMCS]{LASTn:lang}

\section{Type Preservation}
\label{as:LASTn:tp}

\ExecuteMetaData[\fileLMCS]{LASTn:tp}

\section{Translation: Type Preservation}
\label{as:LASTn:transTp}

\ExecuteMetaData[\fileLMCS]{LASTn:transTp}

\section{Operational Correspondence}
\label{as:LASTn:oc}

\ExecuteMetaData[\fileLMCS]{LASTn:oc}

\subsection{Completeness}
\label{as:LASTn:oc:compl}

\ExecuteMetaData[\fileLMCS]{LASTn:oc:compl}

\subsection{Soundness}
\label{as:LASTn:oc:sound}

\ExecuteMetaData[\fileLMCS]{LASTn:oc:sound} \LASTnfalse
\alphalambdatrue \chapter{Appendices for \crtCref{c:alphalambda}}
\label{ac:alphalambda}

This appendix details the translation from the \alcalc into \pibi and proofs of the accompanying operational correspondence results, all presented in \Cref{c:alphalambda}.

\section{Full Translation}
\label{as:alphalambda:trans}

Here we present the full translation of \alcalc typing rules into typed \pibi processes.
In each case, we give the name of the rule along with the full derivation of the translated \pibi process.

\begin{mathparpagebreakable}
    \begin{bussproof}[typ-id]
        \bussAx{
            x:A \vdash \alTrans(x)z = \pFwd [z<>x] \typInf z:A
        }
    \end{bussproof}
    \\
    \begin{bussproof}[typ-cong]
        \bussAssume{
            \Delta \vdash \alTrans(M)z \typInf z:A
        }
        \bussAssume{
            \Delta \equiv \Theta
        }
        \bussBin{
            \Theta \vdash \alTrans(M)z \typInf z:A
        }
    \end{bussproof}
    \\
    \begin{bussproof}[typ-weaken]
        \bussAssume{
            \bunchCtx{\Gamma}[\Delta] \vdash \alTrans(M)z \typInf z:A
        }
        \bussUn{
            \bunchCtx{\Gamma}[\Delta ; \Delta'] \vdash \alTrans(M)z = \pSpw[x->\emptyset | x \in \fv(\Delta')] ; \alTrans(M)z \typInf z:A
        }
    \end{bussproof}
    \\
    \begin{bussproof}[typ-contract]
        \bussAssume{
            \bunchCtx{\Gamma}[\idx{\Delta}{1} ; \idx{\Delta}{2}] \vdash \alTrans(M)z \typInf z:A
        }
        \bussUn{
            \bunchCtx{\Gamma}[\Delta] \vdash \alTrans({M \{ x/\idx{x}{1},x/\idx{x}{2} \mid x \in \fv(\Delta) \}})z = \pSpw[x->x_1,x_2] ; \alTrans(M)z \typInf z:A
        }
    \end{bussproof}
    \\
    \begin{bussproof}[typ-wand-I]
        \bussAssume{
            \Delta , x:A \vdash \alTrans(M)z \typInf z:B
        }
        \bussUn{
            \Delta \vdash \alTrans(\lam x . M)z = \pIn z(x) ; \alTrans(M)z \typInf z:A \wand B
        }
    \end{bussproof}
    \\
    \begin{bussproof}[typ-impl-I]
        \bussAssume{
            \Delta ; x:A \vdash \alTrans(M)z \typInf z:B
        }
        \bussUn{
            \Delta \vdash \alTrans(\alpha x . M)z = \pIn z(x) ; \alTrans(M)z \typInf z:A \to B
        }
    \end{bussproof}
    \\
    \begin{bussproof}[typ-wand-E]
        \bussAssume{
            \Delta \vdash \alTrans(M)x \typInf x:A \wand B
        }
        \bussAssume{
            \Theta \vdash \alTrans(N)y \typInf y:A
        }
        \bussAx{
            x:B \vdash \pFwd [z<>x] \typInf z:B
        }
        \bussBin{
            \Theta , x:A \wand B \vdash \pOut* x[y] ; ( \alTrans(N)y \| \pFwd [z<>x] ) \typInf z:B
        }
        \bussBin{
            \Delta , \Theta \vdash \alTrans(M\ N)z = \pRes{x} \big( \alTrans(M)x \| \pOut* x[y] ; ( \alTrans(N)y \| \pFwd [z<>x] ) \big) \typInf z:B
        }
    \end{bussproof}
    \\
    \begin{bussproof}[typ-impl-E]
        \bussAssume{
            \Delta \vdash \alTrans(M)x \typInf x:A \to B
        }
        \bussAssume{
            \Theta \vdash \alTrans(N)y \typInf y:A
        }
        \bussAx{
            x:B \vdash \pFwd [z<>x] \typInf z:B
        }
        \bussBin{
            \Theta ; x:A \to B \vdash \pOut* x[y] ; ( \alTrans(N)y \| \pFwd [z<>x] ) \typInf z:B
        }
        \bussBin{
            \Delta ; \Theta \vdash \alTrans(M\ N)z = \pRes{x} \big( \alTrans(M)x \| \pOut* x[y] ; ( \alTrans(N)y \| \pFwd [z<>x] ) \big) \typInf z:B
        }
    \end{bussproof}
    \\
    \begin{bussproof}[typ-emp-I]
        \bussAx{
            \mEmpty \vdash \alTrans(\mUnit)z = \pClose z[] \typInf z:\mOne
        }
    \end{bussproof}
    \\
    \begin{bussproof}[typ-true-I]
        \bussAx{
            \aEmpty \vdash \alTrans(\aUnit)z = \pClose z[] \typInf z:\aOne
        }
    \end{bussproof}
    \\
    \begin{bussproof}[typ-emp-E]
        \bussAssume{
            \Delta \vdash \alTrans(M)x \typInf x:\mOne
        }
        \bussAssume{
            \bunchCtx{\Gamma}[\mEmpty] \vdash \alTrans(N)z \typInf z:A
        }
        \bussUn{
            \bunchCtx{\Gamma}[x:\mOne] \vdash \pWait x() ; \alTrans(N)z \typInf z:A
        }
        \bussBin{
            \bunchCtx{\Gamma}[\Delta] \vdash \alTrans(\tLet \mUnit = M \tIn N)z = \pRes{x} \big( \alTrans(M)x \| \pWait x() ; \alTrans(N)z \big) \typInf z:A
        }
    \end{bussproof}
    \\
    \begin{bussproof}[typ-true-E]
        \bussAssume{
            \Delta \vdash \alTrans(M)x \typInf x:\aOne
        }
        \bussAssume{
            \bunchCtx{\Gamma}[\aEmpty] \vdash \alTrans(N)z \typInf z:A
        }
        \bussUn{
            \bunchCtx{\Gamma}[x:\aOne] \vdash \pWait x() ; \alTrans(N)z \typInf z:A
        }
        \bussBin{
            \bunchCtx{\Gamma}[\Delta] \vdash \alTrans(\tLet \aUnit = M \tIn N)z = \pRes{x} \big( \alTrans(M)x \| \pWait x() ; \alTrans(N)z \big) \typInf z:A
        }
    \end{bussproof}
    \\
    \begin{bussproof}[typ-sep-I]
        \bussAssume{
            \Delta \vdash \alTrans(M)y \typInf y:A
        }
        \bussAssume{
            \Theta \vdash \alTrans(N)z \typInf z:B
        }
        \bussBin{
            \Delta , \Theta \vdash \alTrans(\<M,N\>)z = \pOut* z[y] ; \big( \alTrans(M)y \| \alTrans(N)z \big) \typInf z:A \sep B
        }
    \end{bussproof}
    \\
    \begin{bussproof}[typ-conj-I]
        \bussAssume{
            \Delta \vdash \alTrans(M)y \typInf y:A
        }
        \bussAssume{
            \Theta \vdash \alTrans(N)z \typInf z:B
        }
        \bussBin{
            \Delta ; \Theta \vdash \alTrans({(M,N)})z = \pOut* z[y] ; \big( \alTrans(M)y \| \alTrans(N)z \big) \typInf z:A \land B
        }
    \end{bussproof}
    \\
    \begin{bussproof}[typ-sep-E]
        \bussAssume{
            \Delta \vdash \alTrans(M)y \typInf y:A \sep B
        }
        \bussAssume{
            \bunchCtx{\Gamma}[x:A , y:B] \vdash \alTrans(N)z \typInf z:C
        }
        \bussUn{
            \bunchCtx{\Gamma}[y:A \sep B] \vdash \pIn y(x) ; \alTrans(N)z \typInf z:C
        }
        \bussBin{
            \bunchCtx{\Gamma}[\Delta] \vdash \alTrans(\tLet \<x,y\> = M \tIn N)z = \pRes{y} \big( \alTrans(N)z \| \pIn y(x) ; \alTrans(N)z \big) \typInf z:C
        }
    \end{bussproof}
    \\
    \begin{bussproof}[typ-conj-E]
        \bussAssume{
            \Delta \vdash \alTrans(M){x_2} \typInf x_2:A_1 \land A_2
        }
        \bussAx{
            x_1:A_1 \vdash \pFwd [z<>x_1] \typInf z:A_1
        }
        \bussUn{
            x_1:A_1 ; x_2:A_2 \vdash \pSpw[x_2->\emptyset] ; \pFwd [z<>x_1] \typInf z:A_1
        }
        \bussUn{
            x_2:A_1 \land A_2 \vdash \pIn x_2(x_1) ; \pSpw[x_2->\emptyset] ; \pFwd [z<>x_1] \typInf z:A_1
        }
        \bussBin{
            \Delta \vdash \alTrans(\pi_1 M)z = \pRes{x_2} \big( \alTrans(M){x_2} \| \pIn x_2(x_1) ; \pSpw[x_2->\emptyset] ; \pFwd [z<>x_1] \big) \typInf z:A_1
        }
    \end{bussproof}
    \\
    \begin{bussproof}[typ-conj-E]
        \bussAssume{
            \Delta \vdash \alTrans(M){x_2} \typInf x_2:A_1 \land A_2
        }
        \bussAx{
            x_2:A_2 \vdash \pFwd [z<>x_2] \typInf z:A_2
        }
        \bussUn{
            x_1:A_1 ; x_2:A_2 \vdash \pSpw[x_1->\emptyset] ; \pFwd [z<>x_2] \typInf z:A_1
        }
        \bussUn{
            x_2:A_1 \land A_2 \vdash \pIn x_2(x_1) ; \pSpw[x_1->\emptyset] ; \pFwd [z<>x_2] \typInf z:A_2
        }
        \bussBin{
            \Delta \vdash \alTrans(\pi_2 M)z = \pRes{x_2} \big( \alTrans(M){x_2} \| \pIn x_2(x_1) ; \pSpw[x_1->\emptyset] ; \pFwd [z<>x_2] \big) \typInf z:A_2
        }
    \end{bussproof}
    \\
    \begin{bussproof}[typ-disj-I]
        \bussAssume{
            \Delta \vdash \alTrans(M)z \typInf z:A_1
        }
        \bussUn{
            \Delta \vdash \alTrans({\tSel_1(M)})z = \pSelL z ; \alTrans(M)z \typInf z:A_1 \lor A_2
        }
    \end{bussproof}
    \\
    \begin{bussproof}[typ-disj-I]
        \bussAssume{
            \Delta \vdash \alTrans(M)z \typInf z:A_2
        }
        \bussUn{
            \Delta \vdash \alTrans({\tSel_2(M)})z = \pSelR z ; \alTrans(M)z \typInf z:A_1 \lor A_2
        }
    \end{bussproof}
    \\
    \begin{bussproof}[typ-disj-E]
        \def\defaultHypSeparation{\hskip.6ex}
        \def\ScoreOverhang{1pt}
        \bussAssume{
            \Delta \vdash \alTrans(M)x \typInf x:A_1 \lor A_2
        }
        \bussAssume{
            \bunchCtx{\Gamma}[x:A_1] \vdash \alTrans(N_1)z \{ x/x_1 \} \typInf z:C
        }
        \bussAssume{
            \bunchCtx{\Gamma}[x:A_2] \vdash \alTrans(N_2)z \{ x/x_2 \} \typInf z:C
        }
        \bussBin{
            \bunchCtx{\Gamma}[x:A_1 \lor A_2] \vdash \pBraLR x > {\alTrans(N_1)z}{\alTrans(N_2)z} \typInf z:C
        }
        \bussBin{
            \bunchCtx{\Gamma}[\Delta] \vdash \begin{array}[t]{@{}l@{}}
                \alTrans({\tCase M \tOf \{ 1(x_1);N_1 , 2(x_2):N_2 \}})z
                \\
                = \pRes{x} \big( \alTrans(M)x \| \pBraLR x > {\alTrans(N_1)z}{\alTrans(N_2)z} \big) \typInf z:C
            \end{array}
        }
    \end{bussproof}
    \\
    \begin{bussproof}[typ-cut]
        \bussAssume{
            \Delta \vdash \alTrans(M)x \typInf x:A
        }
        \bussAssume{
            \bunchCtx{\Gamma}[x:A] \vdash \alTrans(N)z \typInf z:C
        }
        \bussBin{
            \bunchCtx{\Gamma}[\Delta] \vdash \alTrans(N \{ M/x \})z = \pRes{x} \big( \alTrans(M)x \| \alTrans(N)z \big) \typInf z:C
        }
    \end{bussproof}
\end{mathparpagebreakable}

\section{Operation Correspondence: Proofs}
\label{as:alphalambda:oc}

\subsection{Completeness}
\label{as:alphalambda:oc:completenss}

\ExecuteMetaData[\fileOOPSLA{appendix/translation}]{alphalambda:oc:completeness}

\subsection{Soundness}
\label{as:alphalambda:oc:soundness}

\ExecuteMetaData[\fileOOPSLA{appendix/translation}]{alphalambda:oc:soundness}

 \alphalambdafalse

\mpstAPCPtrue \chapter{Proofs for \crtCref{c:mpstAPCP}}
\label{ac:mpstAPCP}

This appendix details proofs of the results in \Cref{c:mpstAPCP}.

\section{Proofs for \crtCref{s:mpstAPCP:analysis:networks}}
\label{as:mpstAPCP:analysis:networks}

\ExecuteMetaData[\fileLMCS]{mpstAPCP:APCP:charProc}

\pMpstAPCPComplNetsExist*

\ExecuteMetaData[\fileSCICO]{mpstAPCP:proof:complNetsExist}

\subsection{Proofs for \crtCref{s:mpstAPCP:analysis:networks:routerTypes}}
\label{as:mpstAPCP:analysis:networks:routerTypes}

\ExecuteMetaData[\fileSCICO]{mpstAPCP:proof:routerTypes}

\subsection{Proofs for \crtCref{s:mpstAPCP:analysis:transfer}}

\paragraph{Alarm-freedom.}

\tMpstAPCPNoAlarm*

\ExecuteMetaData[\fileSCICO]{mpstAPCP:proof:noAlarm}

\paragraph{Completeness.}

\tMpstAPCPCompleteness*

\ExecuteMetaData[\fileSCICO]{mpstAPCP:proof:completeness}

\paragraph{Soundness.}

We first introduce auxiliary \Cref{l:mpstAPCP:routerNodepEqual,l:mpstAPCP:delayedPrefix}.
Then we prove \Cref{p:mpstAPCP:indep} on \Cpageref{proof:mpstAPCP:indep}.
Finally we prove \Cref{t:mpstAPCP:soundness} on \Cpageref{proof:mpstAPCP:soundness}.

\ExecuteMetaData[\fileSCICO]{mpstAPCP:proof:soundness:1}

\pMpstAPCPIndep*

\ExecuteMetaData[\fileSCICO]{mpstAPCP:proof:indep}

\tMpstAPCPSoundness*

\ExecuteMetaData[\fileSCICO]{mpstAPCP:proof:soundness:2}

\ExecuteMetaData[\fileSCICO]{mpstAPCP:proof:soundness:3}

\section{Proofs for \crtCref{s:mpstAPCP:analysis:mediums}}
\label{as:mpstAPCP:analysis:mediums}

\subsection{Proof for \crtCref{s:mpstAPCP:analysis:mediums:mediumSynth}}
\label{as:mpstAPCP:analysis:medium:mediumSynth}

\tMpstAPCPMediumTypes*

\ExecuteMetaData[\fileSCICO]{mpstAPCP:proof:mediumTypes}

\subsection{Proofs for \crtCref{s:mpstAPCP:analysis:mediums:bisim}}
\label{as:mpstAPCP:analysis:mediums:bisim}

\ExecuteMetaData[\fileSCICO]{mpstAPCP:proof:bisim:1}

\def\fileSCICOBisim{repos/apcp/submissions/SCICO-final/proof_bisim.tex}
\ExecuteMetaData[\fileSCICOBisim]{mpstAPCP:proof:bisim:2}

 \mpstAPCPfalse
\mpstMontrue \chapter{Appendices for \crtCref{c:mpstMon}}
\label{ac:mpstMon}

\section{The Running Example from~\cite{journal/tcs/BocchiCDHY17}}
\label{as:mpstMon:exampleBocchi}

\ExecuteMetaData[\fileRVApp]{mpstMon:exampleBocchi}

\section{A Toolkit for Monitoring Networks of Blackboxes in Practice}
\label{as:mpstMon:toolkit}

\ExecuteMetaData[\fileRVApp]{mpstMon:toolkit}

\section{Relative Types with Locations}
\label{as:mpstMon:relativeTypesWithLocs}

\ExecuteMetaData[\fileRVApp]{mpstMon:relativeTypesWithLocs}

\section{Proof of Soundness}
\label{as:mpstMon:soundnessProof}

\ExecuteMetaData[\fileRVApp]{mpstMon:proof:soundness}

\section{Proof of Transparency}
\label{as:mpstMon:transparencyProof}

\ExecuteMetaData[\fileRVApp]{mpstMon:proof:transparency}

 \mpstMonfalse

\thumbfalse

\fi 

%
%
%
%
%
%
%
%

\newcommand*{\promitem}[4]{\noindent \textbf{#1}. \emph{#2}. #3.~\mbox{#4}\medskip}

\clearpage \pagestyle{empty}

\setlength{\columnsep}{2em}
\begin{multicols}{2}
        [\subsection*{Titles in the IPA Dissertation Series since 2021}]

\promitem{D. Frumin}
         {Concurrent Separation Logics for Safety, Refinement, and
Security}
         {Faculty of Science, Mathematics and Computer Science, RU}
		 {2021-01}

\promitem{A. Bentkamp}
         {Superposition for Higher-Order Logic}
         {Faculty of Sciences, Department of Computer Science, VU}
         {2021-02}

\promitem{P. Derakhshanfar}
         {Carving Information Sources to Drive Search-based Crash Reproduction and Test Case Generation}
         {Faculty of Electrical Engineering, Mathematics, and Computer Science, TUD}
         {2021-03}

\promitem{K. Aslam}
         {Deriving Behavioral Specifications of Industrial Software Components}
         {Faculty of Mathematics and Computer Science, TU/e}
         {2021-04}

\promitem{W. Silva Torres}
         {Supporting Multi-Domain Model Management}
         {Faculty of Mathematics and Computer Science, TU/e}
         {2021-05}

\promitem{A. Fedotov}
         {Verification Techniques for xMAS}
         {Faculty of Mathematics and Computer Science, TU/e}
         {2022-01}

\promitem{M.O. Mahmoud}
         {GPU Enabled Automated Reasoning}
         {Faculty of Mathematics and Computer Science, TU/e}
         {2022-02}

\promitem{M. Safari}
         {Correct Optimized GPU Programs}
         {Faculty of Electrical Engineering, Mathematics \& Computer Science, UT}
         {2022-03}

\promitem{M. Verano Merino}
         {Engineering Language-Parametric End-User Programming Environments for DSLs}
         {Faculty of Mathematics and Computer Science, TU/e}
         {2022-04}

\promitem{G.F.C. Dupont}
         {Network Security Monitoring in Environments where Digital and Physical Safety are Critical}
         {Faculty of Mathematics and Computer Science, TU/e}
         {2022-05}
		 
\promitem{T.M. Soethout}
         {Banking on Domain Knowledge for Faster Transactions}
         {Faculty of Mathematics and Computer Science, TU/e}
         {2022-06}
		
\promitem{P. Vukmirovi\'{c}}
         {Implementation of Higher-Order Superposition}
         {Faculty of Sciences, Department of Computer Science, VU}
         {2022-07}

\promitem{J. Wagemaker}
         {Extensions of (Concurrent) Kleene Algebra}
         {Faculty of Science, Mathematics and Computer Science, RU}
		 {2022-08}
		 
\promitem{R. Janssen}
         {Refinement and Partiality for Model-Based Testing}
         {Faculty of Science, Mathematics and Computer Science, RU}
		 {2022-09}

\promitem{M. Laveaux}
         {Accelerated Verification of Concurrent Systems}
         {Faculty of Mathematics and Computer Science, TU/e}
         {2022-10}
		 
\promitem{S. Kochanthara}
         {A Changing Landscape: On Safety \& Open Source in Automated and Connected Driving}
         {Faculty of Mathematics and Computer Science, TU/e}
         {2023-01}
		 
\promitem{L.M. Ochoa Venegas}
         {Break the Code? Breaking Changes and Their Impact on Software Evolution}
         {Faculty of Mathematics and Computer Science, TU/e}
         {2023-02}

\promitem{N. Yang}
         {Logs and models in engineering complex embedded production software systems}
         {Faculty of Mathematics and Computer Science, TU/e}
         {2023-03}
		 
\promitem{J. Cao}
         {An Independent Timing Analysis for Credit-Based Shaping in Ethernet TSN}
         {Faculty of Mathematics and Computer Science, TU/e}
         {2023-04}

\promitem{K. Dokter}
         {Scheduled Protocol Programming}
         {Faculty of Mathematics and Natural Sciences, UL}
         {2023-05}

\promitem{J. Smits}
         {Strategic Language Workbench Improvements}
         {Faculty of Electrical Engineering, Mathematics, and Computer Science, TUD}
         {2023-06}

\promitem{A. Arslanagi\'{c}}
         {Minimal Structures for Program Analysis and Verification}
         {Faculty of Science and Engineering, RUG}
         {2023-07}

\promitem{M.S. Bouwman}
         {Supporting Railway Standardisation with Formal Verification}
         {Faculty of Mathematics and Computer Science, TU/e}
         {2023-08}

\promitem{S.A.M. Lathouwers}
         {Exploring Annotations for Deductive Verification}
         {Faculty of Electrical Engineering, Mathematics \& Computer Science, UT}
         {2023-09}
		 
\promitem{J.H. Stoel}
         {Solving the Bank, Lightweight Specification and Verification Techniques for Enterprise Software}
         {Faculty of Mathematics and Computer Science, TU/e}
         {2023-10}

\promitem{D.M. Groenewegen}
         {WebDSL: Linguistic Abstractions for Web Programming}
         {Faculty of Electrical Engineering, Mathematics, and Computer Science, TUD}
         {2023-11}

\promitem{D.R. do Vale}
         {On Semantical Methods for Higher-Order Complexity Analysis}
         {Faculty of Science, Mathematics and Computer Science, RU}
         {2024-01}

\promitem{M.J.G Olsthoorn}
         {More Effective Test Case Generation with Multiple Tribes of AI}
         {Faculty of Electrical Engineering, Mathematics, and Computer Science, TUD}
         {2024-02}

\promitem{B. van den Heuvel}
         {Correctly Communicating Software: Distributed, Asynchronous, and Beyond}
         {Faculty of Science and Engineering, RUG}
         {2024-03}

\end{multicols}

\end{document}


\maketitle

\begin{abstract}
    \noindent
    We tackle the challenge of ensuring the deadlock-freedom property for message-passing processes that communicate asynchronously in cyclic process networks. 
    Our contributions are two-fold.
    First, we present Asynchronous Priority-based Classical Processes (\APCP), a session-typed process framework that supports asynchronous communication, delegation, and recursion in cyclic process networks.
    Building upon the Curry-Howard correspondences between linear logic and session types, we establish essential meta-theoretical results for \APCP, most notably deadlock-freedom.
    Second, we present a new concurrent \lamcalc with asynchronous session types, dubbed \LASTn.
    We illustrate \LASTn by example and establish its meta-theoretical results;
    in particular, we show how to soundly transfer the deadlock-freedom guarantee from \APCP. To this end, we develop a translation of terms in \LASTn into  processes in \APCP that satisfies  a strong formulation of operational correspondence.
\end{abstract}


\section{Introduction}
\label{s:LASTn:intro}

Modern software systems often comprise independent components that coordinate by exchanging messages.
The $\pi$-calculus~\cite{journal/ic/MilnerPW92,book/Milner89} is a mature formalism for specifying and reasoning about message-passing processes; in particular, it offers a rigorous foundation for designing \emph{type systems} that statically enforce communication correctness.
A well-known approach in this line is \emph{session types}~\cite{conf/concur/Honda93,conf/esop/HondaVK98,conf/secret/YoshidaV07}, which specify the structure of the two-party communication protocols implemented by the channels of a process.
In this \ifthesis chapter\else paper\fi, we are interested in session types as a governing discipline in concurrent and functional paradigms, in conjunction with two important aspects of message-passing concurrency, namely the \emph{network topologies} formed by interacting processes and the underlying discipline of \emph{asynchronous communication}.

The study of session types has gained a considerable impulse after the discovery by Caires and Pfenning~\cite{conf/concur/CairesP10} and Wadler~\cite{conf/icfp/Wadler12} of Curry-Howard correspondences between session types and linear logic~\cite{journal/tcs/Girard87}.
The present work is motivated by (and develops further) two salient features of type systems derived from these correspondences, namely (i)~their remarkably effective approach to establishing the \emph{deadlock-freedom} property for processes, and (ii)~their clear connections with \emph{functional calculi} with concurrency. These two aspects form the central themes of this \ifthesis chapter\else paper\fi, and we discuss them in order.

\paragraph{Deadlock-freedom.}
Curry-Howard approaches to session types   induce a very precise form of interaction between parallel processes: they interpret the \ruleLabel{cut}-rule of linear logic as the interaction of two processes on \emph{exactly one channel}.
While this design elegantly rules out the insidious circular dependencies that lead to deadlocks, there is a catch: typable processes cannot be connected to form \emph{cyclic} network topologies---only tree-shaped networks are allowed.
Hence, type systems based upon Curry-Howard foundations reject whole classes of  process networks that are cyclic but deadlock-free~\cite{conf/express/DardhaP15,journal/jlamp/DardhaP22}.
This includes important concurrency patterns, such as those exemplified by Milner's cyclic scheduler~\cite{book/Milner89}.

The problem of designing type systems that rule out circular dependencies and deadlocks while allowing for cyclic topologies has received considerable attention.
Works by Kobayashi and others  have developed advanced solutions; see, e.g.,~\cite{conf/concur/Kobayashi06,conf/lics/Padovani14,conf/fossacs/DardhaG18}.
In a nutshell, these works exploit orderings based on {priority annotations} on types to detect and avoid circular dependencies.
Dardha and Gay~\cite{conf/fossacs/DardhaG18} have shown how to incorporate this priority-based approach in the realm of session type systems based on linear logic; it boils down to replacing the \ruleLabel{cut}-rule with a \ruleLabel{cycle}-rule and adding priority checks in other selected typing rules.
Their work thus extends the class of typable processes to cover cyclic network topologies, while retaining strong ties with linear logic.

Unfortunately, none of the methods proposed until now consider \emph{session types with asynchronous communication}.
Addressing asynchronous communication is of clear practical relevance: it is the standard in most distributed systems and web-based applications nowadays.
In a process calculi setting, asynchronous communication means that output prefixes are non-blocking~\cite{conf/ecoop/HondaT91,conf/occ/HondaT91,report/Boudol92}, and that exchanged messages implicitly or explicitly reside in an auxiliary structure, such as a buffer or a queue~\cite{chapter/BeauxisPV08}.
In the context of session types, asynchrony moreover means that the ordering of messages \emph{within a session} should be respected, but messages from \emph{different sessions} need not be ordered~\cite{conf/forte/KouzapasYH11}.

To address this gap,
in the first part of the \ifthesis chapter \else paper \fi
we define a new session-typed asynchronous $\pi$-calculus, \APCP \emph{(Asynchronous Priority-based Classical Processes)}, for which we develop its fundamental meta-theoretical results.
The design of \APCP builds upon insights developed in several prior works:
\begin{itemize}
    \item
        Advanced type systems that exploit annotations on types to enforce deadlock-freedom of cyclic process networks, already mentioned;

    \item
        Dardha and Gay's \PCP (Priority-based Classical Processes)~\cite{conf/fossacs/DardhaG18}, also already mentioned, which incorporates into Wadler's \CP~\cite{conf/icfp/Wadler12} (Classical Processes; derived from classical linear logic) Padovani's simplification of Kobayashi's type annotations~\cite{conf/lics/Padovani14}.

    \item
        DeYoung \etal's asynchronous semantics for session communication, extending the correspondence between intuitionistic linear logic and session types~\cite{conf/csl/DeYoungCPT12}.
\end{itemize}
Our calculus
\APCP combines these semantics for asynchronous communication with \PCP's  priority-based type system.
The design of \APCP uncovers fundamental properties of type systems for asynchronous communication.
A particular insight is the following: because outputs are non-blocking, \APCP simplifies priority management while preserving deadlock-freedom.
Additionally, as an orthogonal feature, \APCP increases expressivity by supporting tail-recursion without compromising deadlock-freedom.
We motivate these features of \APCP by discussing Milner's cyclic scheduler in \Cref{s:APCP:milner}.

\paragraph{Functional calculi with concurrency.}
Session types are paradigm-independent, in the sense that they can be accommodated on top of programming models and languages in different paradigms---concurrent, object-oriented, and functional.
In the functional setting, a milestone is the asynchronous concurrent \lamcalc with sessions by Gay and Vasconcelos~\cite{journal/jfp/GayV10}, which in the following we call \LAST (Linear Asynchronous Session Types).
\LAST is a call-by-value calculus with buffered channels governed by sessions, whose type system ensures that well-typed configurations (collections of threads) respect their ascribed protocols (protocol fidelity) but does not guarantee deadlock-freedom, i.e., it permits cyclic  configurations with circular dependencies.

\LAST has been brought back to the spotlight through Wadler's work on \GV (Good Variation), a synchronous variant of \LAST without cyclic configurations.
Terms in \GV are guaranteed to be deadlock-free via a (typed) translation into \CP~\cite{conf/icfp/Wadler12}.
Subsequently, Kokke and Dardha formulated \PGV (Priority-based \GV), an extension of \GV that strictly augments the class of deadlock-free computations with cyclic configurations by leveraging priorities~\cite{conf/forte/KokkeD21,report/KokkeD21}, following Padovani and Novara~\cite{conf/forte/PadovaniN15}.

\ifLASTn In this chapter\else In the second part of the paper\fi, we study asynchronous, deadlock-free communication with support for cyclic topologies in the  setting of a prototypical functional programming language.
We present \LASTn, a new call-by-name variant of \LAST.
Notably, communication in \LASTn is asynchronous---a feature not accounted for by \GV and \PGV.

We equip \LASTn with a deliberately simple type system, with functional and session types, which ensures type preservation/protocol fidelity but not deadlock-freedom (just like the type system for \LAST).
To address this gap, we develop a way of soundly transferring the deadlock-freedom property from \APCP to \LASTn.
This transference of results hinges on a translation of \LASTn programs into \APCP processes, in the style for Milner's seminal work~\cite{conf/icalp/Milner90,journal/mscs/Milner92}. The translation clarifies the role of \APCP as an abstract model for asynchronous, functional concurrency; it  satisfies  in particular a tight form of \emph{operational correspondence} that follows the well-known formulation by Gorla~\cite{journal/ic/Gorla10}.
This way, we can ensure that a (class of) well-typed \LASTn  programs with cyclic configurations satisfies deadlock-freedom.
While the development of \LASTn  is of interest in itself (it improves over \GV and \PGV, as just discussed), it is also a significant test for \APCP, its expressiveness and meta-theoretical results.

\paragraph{Contributions}
In summary, in this \ifthesis chapter \else paper \fi we make the following contributions:
\begin{enumerate}
    \item
        The process calculus \APCP, its associated type system, and its essential meta-theoretical properties: type preservation (\Cref{t:APCP:sr}) and deadlock-freedom (\Cref{t:APCP:df}).

    \item
        The functional calculus \LASTn, its associated type system, and its meta-theoretical property of type preservation (\Cref{t:LASTn:LASTn:tp}).

    \item
        A translation of \LASTn into \APCP that enjoys operational correspondence properties (\Cref{t:LASTn:LASTn:transCompl,t:LASTn:LASTn:transSound}) and unlocks the transfer of deadlock-freedom from \APCP to \LASTn (\Cref{t:LASTn:LASTn:df}).
\end{enumerate}

\ifthesis
In summary, in this chapter I discuss the process calculus \APCP, its associated type system, and its essential meta-theoretical properties: type preservation (\Cref{t:APCP:sr}) and deadlock-freedom (\Cref{t:APCP:df}).
\fi

\ifthesis

In summary, in this this chapter I discuss the functional calculus \LASTn, its associated type system, and a translation of \LASTn into \APCP that enjoys operational correspondence properties (\Cref{t:LASTn:LASTn:transCompl,t:LASTn:LASTn:transSound}) and transfers from \APCP to \LASTn a form of deadlock-freedom (\Cref{t:LASTn:LASTn:df}).
\fi

\paragraph{Organization.}
In \Cref{s:LASTn:examples}, we motivate \APCP  and  \LASTn    by example.
\Cref{s:LASTn:APCP} defines the language of \APCP processes and its type system, and establishes its meta-theoretical properties.
\Cref{s:LASTn:LAST} recalls \LAST and its type system, as proposed in \cite{journal/jfp/GayV10}, and briefly discusses the issue of devising an operationally correct translation into \APCP. 
Building upon this background, \Cref{s:LASTn:LASTn} presents \LASTn as a call-by-name variant of \LAST, develops its meta-theoretical results, and gives a correct translation of \LASTn into \APCP. 
\Cref{s:LASTn:APCP:rw} discusses related work and \Cref{s:LASTn:APCP:concl} draws conclusions.
\ifappendix\Cref{as:LASTn} \else\cite{thesis/vdHeuvel24ex} \fi collects omitted definitions and proofs for \LASTn.

\ifthesis
\paragraph{Organization.}
In \Cref{s:APCP:milner}, we motivate \APCP by introducing Milner's cyclic scheduler.
\Cref{s:APCP:APCP} defines the language of \APCP processes and its type system, establishes its properties, and revisits Milner's cyclic scheduler from a typed perspective.
\Cref{s:APCP:rw} discusses related work and \Cref{s:APCP:concl} draws conclusions.
\ifappendix \Cref{ac:APCP} \else \cite{thesis/vdHeuvel24ex} \fi collects omitted proofs.
\paragraph{Organization.}

In \Cref{s:LASTn:example}, I motivate \LASTn by example.
\Cref{s:LASTn:LAST} recalls \LAST and its type system, as proposed in \cite{journal/jfp/GayV10}, and briefly discusses the issue of devising an operationally correct translation into \APCP. 
Building upon this background, \Cref{s:LASTn:LASTn} presents \LASTn as a call-by-name variant of \LAST, develops its meta-theoretical results, and gives a correct translation of \LASTn into \APCP. 
\Cref{s:LASTn:rw} discusses related work and \Cref{s:LASTn:concl} draws conclusions.
\ifappendix\Cref{ac:LASTn} \else\cite{thesis/vdHeuvel24ex} \fi collects omitted definitions and proofs for \LASTn.
\fi

\paragraph{Origin of the results.}
The current \ifthesis chapter \else paper \fi combines, revises, and extends our papers~\cite{conf/ice/vdHeuvelP21,conf/express/vdHeuvelP22}.
Here we provide fully detailed proofs, expanded examples, and consolidated comparisons of related works.
While the paper~\cite{conf/ice/vdHeuvelP21} offered an abridged introduction to \APCP, the paper 
\cite{conf/express/vdHeuvelP22} developed \LASTn, a functional calculus with concurrency (also in the style of Gay and Vasconcelos's \LAST), and gave a correct translation into \APCP. 
Novelties with respect to~\cite{conf/ice/vdHeuvelP21} include a revised treatment of type preservation and recursion in \APCP.
With respect to \cite{conf/express/vdHeuvelP22}, in this presentation we have revised \LASTn into \LASTn, in order to showcase a more clear connection with \LAST and a well-known reduction strategy (call-by-name). 
As a result, \Cref{s:LASTn:LASTn,s:LASTn:LAST} are entirely new to this presentation.

\section{Motivating Examples}
\label{s:LASTn:examples}

In this section, we informally describe \APCP and \LASTn, the two calculi presented in this \ifthesis chapter\else paper\fi, using motivating examples that illustrate their distinctive features and expressiveness.

\subsection{Milner's Cyclic Scheduler in \APCP}
\label{s:APCP:milner}

We motivate \APCP by  considering Milner's cyclic scheduler~\cite{book/Milner89}, a recursive process that   relies on a cyclic network and asynchronous communications; it is inspired by Dardha and Gay~\cite{conf/fossacs/DardhaG18}, who use \PCP to type a synchronous, non-recursive version of the scheduler.

\begin{figure}[t]
    \begin{minipage}[b]{.6\textwidth}
        \begin{tikzpicture}[
            proc/.style={fill=none, draw=none, shape=circle, inner sep=-.5},
            thick/.style={-, line width=1pt},
            thicker/.style={-, line width=1.2pt},
            chan/.style={-, thick, inner sep=3mm},
            chanb/.style={chan, inner sep=2mm},
            chant/.style={chanb, thicker},
            aChan/.style={color=aChanClr},
            bChan/.style={color=bChanClr},
            cChan/.style={color=cChanClr},
            dChan/.style={color=dChanClr},
        ]

            \node [proc] (A1) at (-4, 5.25) {$A_1$};
            \node [proc] (A2) at (-2, 6.5) {$A_2$};
            \node [proc] (A3) at (0, 5.25) {$A_3$};
            \node [proc] (A4) at (0, 3.25) {$A_4$};
            \node [proc] (A5) at (-2, 2) {$A_5$};
            \node [proc] (A6) at (-4, 3.25) {$A_6$};
            \node [proc] (P1) at (-5.5, 6) {$P_1$};
            \node [proc] (P2) at (-2, 8) {$P_2$};
            \node [proc] (P3) at (1.5, 6) {$P_3$};
            \node [proc] (P4) at (1.5, 2.5) {$P_4$};
            \node [proc] (P5) at (-2, .5) {$P_5$};
            \node [proc] (P6) at (-5.5, 2.5) {$P_6$};

            \draw [chan, bend left, looseness=0.75] (A1) to node [dChan, below, pos=.3] {$d_1$} (A2);
            \draw [chanb, bend left, looseness=0.75] (A1) to node [cChan, below, pos=.9] {$c_1$} (A2);
            \draw [chanb, bend left, looseness=0.75] (A2) to node [dChan, below, pos=.15] {$d_2$} (A3);
            \draw [chan, bend left, looseness=0.75] (A2) to node [cChan, below, pos=.7] {$c_2$} (A3);
            \draw [chanb, bend left=15, looseness=0.75] (A3) to node [dChan, left, pos=.2] {$d_3$} (A4);
            \draw [chanb, bend left=15, looseness=0.75] (A3) to node [cChan, left, pos=.8] {$c_3$} (A4);
            \draw [chan, bend left, looseness=0.75] (A4) to node [dChan, above, pos=.3] {$d_4$} (A5);
            \draw [chanb, bend left, looseness=0.75] (A4) to node [cChan, above, pos=.9] {$c_4$} (A5);
            \draw [chanb, bend right=330, looseness=0.75] (A5) to node [dChan, above, pos=.1] {$d_5$} (A6);
            \draw [chan, bend right=330, looseness=0.75] (A5) to node [cChan, above, pos=.7] {$c_5$} (A6);
            \draw [chanb, bend left=15, looseness=0.75] (A6) to node [dChan, right, pos=.15] {$d_6$} (A1);
            \draw [chanb, bend left=15, looseness=0.75] (A6) to node [cChan, right, pos=.8] {$c_6$} (A1);
            \draw [chan] (A1) to node [aChan, below, pos=.3] {$a_1$} (P1);
            \draw [chan] (A1) to node [bChan, below, pos=1] {$b_1$} (P1);
            \draw [chant] (A2) to node [aChan, left, pos=.1] {$a_2$} (P2);
            \draw [chant] (A2) to node [bChan, left, pos=.9] {$b_2$} (P2);
            \draw [chan] (A3) to node [aChan, above, pos=0] {$a_3$} (P3);
            \draw [chan] (A3) to node [bChan, above, pos=.7] {$b_3$} (P3);
            \draw [chan] (A4) to node [aChan, above, pos=.3] {$a_4$} (P4);
            \draw [chan] (A4) to node [bChan, above, pos=1] {$b_4$} (P4);
            \draw [chant] (A5) to node [aChan, right, pos=.1] {$a_5$} (P5);
            \draw [chant] (A5) to node [bChan, right, pos=.9] {$b_5$} (P5);
            \draw [chan] (A6) to node [aChan, below, pos=0] {$a_6$} (P6);
            \draw [chan] (A6) to node [bChan, below, pos=.7] {$b_6$} (P6);
        \end{tikzpicture}
    \end{minipage}%
    \hfill%
    \begin{minipage}[b]{.35\textwidth}
        \caption{
            \label{f:APCP:milner}
            Milner's cyclic sche\-duler with 6 workers.
            Lines denote channels connecting processes on the indicated names.
        }
    \end{minipage}
\end{figure}

\paragraph{The scheduler, informally.}
The scheduler consists of $n \geq 1$ worker processes $P_i$ (the workers, for short), each attached to a partial scheduler $A_i$.
The partial schedulers connect to each other in a ring structure, together forming the cyclic scheduler.
Connections consist of pairs of endpoints; we further refer to these endpoints by the names that represent them.

The scheduler then lets the workers perform their tasks in rounds, each new round triggered by the leading partial scheduler $A_1$ (the \emph{leader}) once each worker finishes their previous task.
We refer to the non-leading partial schedulers $A_{i+1}$ for $1 \leq i < n$ as the \emph{followers}.

\Cref{f:APCP:milner} illustrates the process network of Milner's cyclic scheduler with 6 workers ($n=6$).
Each partial scheduler $A_i$ has a name \aChan{$a_i$} to connect with the worker $P_i$'s name \bChan{$b_i$}.
The leader $A_1$ has a name \cChan{$c_n$} to connect with $A_n$ and a name \dChan{$d_1$} to connect with $A_2$ (or with $A_1$ if $n=1$; we further elide this case for brevity).
Each follower $A_{i+1}$ has a name \cChan{$c_i$} to connect with~$A_i$ and a name \dChan{$d_{i+1}$} to connect with $A_{i+2}$ (or with~$A_1$ if $i+1=n$; we also elide this case).

In each round, each follower $A_{i+1}$ awaits a start signal from $A_i$, and then asynchronously signals $P_{i+1}$ and $A_{i+2}$ to start.
After awaiting acknowledgment from $P_{i+1}$ and a next round signal from $A_i$, the follower then signals next round to $A_{i+2}$.
The leader $A_1$, which starts each round of tasks, signals $A_2$ and $P_1$ to start, and, after awaiting acknowledgment from~$P_1$, signals a next round to $A_2$.
Then, the leader awaits $A_n$'s start and next round signals.
It is crucial that~$A_1$ does not await $A_n$'s start signal before starting $P_1$, as the leader would otherwise not be able to initiate rounds of tasks.

\paragraph{Syntax of \APCP.}
Before formally specifying the scheduler, we briefly introduce the syntax of \APCP, and how asynchronous communication works; \ifthesis\Cref{s:APCP:APCP} \else\Cref{s:APCP} \fi gives formal definitions.

Let us write $a,b,c,\ldots,x, y, z, \ldots$ to denote names.
We write $\pOut x[a,b]$ and $\pIn y(w,z) ; P$ to denote processes for sending and receiving, respectively.
The names $a$ and $w$ stand for the payloads (i.e., the names that represent the value being transmitted), whereas $b$ and $z$ stand for \emph{continuations}, i.e., the names on which the rest of the session should be performed.
This \emph{continuation-passing style} for asynchronous communication is required to ensure the correct ordering of messages within a session.
A communication redex in \APCP is thus of the form
$\pRes{xy} ( \pOut x[a,b] \| \pIn y(w,z) ; P )$:
the restriction $\pRes{xy}$ serves to declare that $x$ and $y$ are dual names of the same channel, and $\cdot \|\cdot$ denotes parallel composition.
This process reduces to $P \{a/w, b/z\}$: $P$'s names $w$ and $z$ are substituted for by $a$ and $b$, respectively.
Notice how, since the send is a standalone  process, it cannot block any other communications in the process.

The communication of labels $\ell, \ell', \ldots$ follows the same principle: the process $\pSel x[b] < \ell$ denotes the output of a label $\ell$ on $x$, and the process $\pBra y(z) > \ell ; P$ blocks until a label $\ell$ is received on $y$ before continuing as $P$; here again, $b$ and $z$ are continuation names that are sent and received together with $\ell$.
Finally, process $\pRec X(\tilde{x}) ; P$ denotes a recursive definition.
Here, $P$ has access to the names in $\tilde{x}$ and may contain recursive calls $\pCall X<\tilde{y}>$ to indicate a repition of $P$.
Upon such a recursive call $\pCall X<\tilde{y}>$, the names in $\tilde{y}$ are assigned to $\tilde{x}$ in the next round of $P$.

\paragraph{The scheduler in \APCP.}
We now formally specify the partial schedulers.
Because each sent label requires a restriction to bind the selection's continuation name to the rest of the session, these processes may look rather complicated.
For example, (part of) the leader is specified as follows (where $\dChan{d'_1},\aChan{a''_1}$ are used in the omitted remainder):
\[
    A_1 \deq \pRec X(\aChan{a_1},\cChan{c_n},\dChan{d_1}) ; \pRes{a\dChan{d'_1}} ( \pSel \dChan{d_1}[a] < \sff{start} \| \pRes{b\aChan{a'_1}} ( \pSel \aChan{a_1}[b] < \sff{start} \| \pBra \aChan{a'_1}(\aChan{a''_1}) > \sff{ack} ; \ldots ) )
\]
To improve readability, we rely on the notation $\pSel* x < \ell \cdot P$, which is syntactic sugar that elides the continuation name involved (cf.\ \Cref{n:APCP:sugar}); our use of the floating dot `\,$\cdot$\,' is intended to stress that the communication of $\ell$ along $x$ does not block any further prefixes in $P$, and the overline `\,$\ol{\phantom{x}}$\,' to indicate that there is a hidden restriction.
Similarly, the notation $\pBra* x > \ell ; P$ elides the continuation name in receiving labels, though it remains blocking hence the `\,$;$\,'.

The leader and followers are then specified as follows:
\def\cyclicAsyncMilnerScheds{%
    \begin{align*}
        A_1 &\deq \pRec X(\aChan{a_1},\cChan{c_n},\dChan{d_1}) ; \begin{array}[t]{@{}l@{}}
            \pSel* \dChan{d_1} < \sff{start} \cdot \pSel* \aChan{a_1} < \sff{start} \cdot \pBra* \aChan{a_1} > \sff{ack} ; \pSel* \dChan{d_1} < \sff{next} \cdot {}
            \\
            \pBra* \cChan{c_n} > \sff{start} ; \pBra* \cChan{c_n} > \sff{next} ; \pCall X<\aChan{a_1},\cChan{c_n},\dChan{d_1}>
        \end{array}
        \\
        A_{i+1} &\deq \pRec X(\aChan{a_{i+1}},\cChan{c_i},\dChan{d_{i+1}}) ; \begin{array}[t]{@{}l@{}}
            \pBra* \cChan{c_i} > \sff{start} ; \pSel* \aChan{a_{i+1}} < \sff{start} \cdot \pSel*\dChan{d_{i+1}} < \sff{start} \cdot \pBra* \aChan{a_{i+1}} > \sff{ack} ; {}
            \\
            \pBra* \cChan{c_i} > \sff{next} ; \pSel* \dChan{d_{i+1}} < \sff{next} \cdot \pCall X<\aChan{a_{i+1}},\cChan{c_i},\dChan{d_{i+1}}>
        \end{array}
        \qquad \forall 1 \leq i < n
    \end{align*}
}
\cyclicAsyncMilnerScheds

\noindent
Assuming that each worker $P_i$ is specified such that it behaves as expected on the name~\bChan{$b_i$}, we formally specify the scheduler as a ring of partial schedulers connected to workers:
\def\cyclicAsyncMilnerSched{%
    \[
        Sched_n \deq \pRes{\cChan{c_1}\dChan{d_1}} \ldots \pRes{\cChan{c_n}\dChan{d_n}} \big( \pRes{\aChan{a_1}\bChan{b_1}} ( A_1 \| P_1 ) \| \ldots \| \pRes{\aChan{a_n}\bChan{b_n}} ( A_n \| P_n ) \big)
    \]
}
\cyclicAsyncMilnerSched
We return to this example in \Cref{s:APCP:APCP:milner}, where we type check the scheduler using \APCP to show that it is deadlock-free (which is not obvious from its process definition).

\subsection{A Bookshop Scenario in \LASTn}
\label{s:LASTn:example}

Our new calculus \LASTn is a call-by-name variant of Gay and Vasconcelos' \LAST~\cite{journal/jfp/GayV10} with linear resources.
We briefly motivate the design of \LASTn by adapting the running example of~\cite{journal/jfp/GayV10}, which involves  a mother interacting with a  bookshop to buy a book for her son.

First, we define a term representing the shop.
The shop has an endpoint $\term{s}$ on which it communicates with a client.
First, the shop receives a book title and then offers a choice between buying the book or only accessing its blurb (the text on the book's back cover).
If the client decides to buy, the shop receives credit card information and sends the book to the client.
Otherwise, if the client requests the blurb, the shop sends its text.
In \LASTn, we can define this shop as follows:
\[
    \term{\sff{Shop}_s} \deq \term{
        \begin{array}[t]{@{}l@{}}
            \tLet (\textit{title},s_1) = \tRecv s \tIn
            \\
            \tCase s_1 \tOf \{
                \begin{array}[t]{@{}l@{}}
                    \sff{buy}: \lam s_2. \begin{array}[t]{@{}l@{}}
                        \tLet (\textit{card},s_3) = \tRecv s_2 \tIn
                        \\
                        \tLet s_4 = \tSend \text{book}(\textit{title}) \, s_3 \tIn
                        \\
                        \tClose s_4 ; ()
                        ,
                    \end{array}
                    \\
                    \sff{blurb}: \lam s_2. \begin{array}[t]{@{}l@{}}
                        \tLet s_3 = \tSend \text{blurb}(\textit{title}) \, s_2 \tIn
                        \\
                        \tClose s_3 ; ()
                    \}
                \end{array}
            \end{array}
        \end{array}
    }
\]
The functional behavior of \LASTn is standard, so we only explain the message-passing components in the above term:
\begin{itemize}

    \item
        `$\term{\tRecv s}$' waits for a message to be ready on $\term{s}$, and returns a pair containing the message and an endpoint on which to continue the session.

    \item
        `$\term{\tCase s \tOf \{ \ldots \}}$' waits for a label to be ready on $\term{s}$, determining a continuation.
        Continuations are defined as abstractions, which will be applied to the session's continuation once the label has been received.

    \item
        `$\term{\tSend M \, s}$' buffers the message $M$ on endpoint $\term{s}$, which can be received asynchronously on a connected endpoint.

    \item
        `$\term{\tClose s ; M}$' waits until the session on the endpoint $\term{s}$ can be closed.

\end{itemize}

Next, we define a term abstracting the son's behavior.
Term $\term{\sff{Son}_{s',m'}}$, given below, has an endpoint $\term{s'}$ on which he communicates with the shop, and an endpoint $\term{m'}$ for communication with his mother.
The son sends a book title to the shop, and then selects to buy.
Notation `$\term{\tSel \ell \, \sff{s}}$' denotes the buffering of the label $\term{\ell}$ on endpoint $\term{s}$, to be received asynchronously on a connected endpoint.
After this selection, to let his mother pay for him, the son proceeds to send the shop's endpoint to his mother.
Finally, he receives the book from his mother, and returns it as a result of the computations.
\[
    \term{\sff{Son}_{s',m'}} \deq \term{
        \begin{array}[t]{@{}l@{}}
            \tLet s'_1 = \tSend \text{``Dune''} \, s' \tIn
            \\
            \tLet s'_2 = \tSel \sff{buy} \, s'_1 \tIn
            \\
            \tLet m'_1 = \tSend s'_2 \, m' \tIn
            \\
            \tLet (\textit{book},m'_2) = \tRecv m'_1 \tIn
            \\
            \tClose m'_2 ; \textit{book}
        \end{array}
    }
\]

Finally, we define the mother, who has an endpoint $\term{m}$ to communicate with her son.
She receives the shop endpoint from her son, and then sends her credit card information to the shop.
She then receives the book from the shop, and sends it to her son.
\[
    \term{\sff{Mother}_m} \deq \term{
        \begin{array}[t]{@{}l@{}}
            \tLet (x,m_1) = \tRecv m \tIn
            \\
            \tLet x_1 = \tSend \text{visa} \, x \tIn
            \\
            \tLet (\textit{book},x_2) = \tRecv x_1 \tIn
            \\
            \tLet m_2 = \tSend \textit{book} \, m_1 \tIn
            \\
            \tClose m_2 ; \tClose x_2 ; ()
        \end{array}
    }
\]

Now, we only have to compose these terms together, connecting all the endpoints appropriately.
This is achieved by term $\term{\sff{Sys}}$ below, which relies on two additional constructs.
\begin{itemize}

    \item
        `$\term{\tNew}$' creates a channel, and returns the channel's pair of endpoints; asynchronous communication is achieved by connecting the endpoints through an ordered buffer.

    \item
        `$\term{\tSpawn M ; N}$' launches $\term{M}$ as a \emph{thread}, running concurrently to the continuation $\term{N}$.

\end{itemize}
Concretely, we have:
\[
    \term{\sff{Sys}} \deq \term{
        \begin{array}[t]{@{}l@{}}
            \tLet (s,s') = \tNew \tIn \tSpawn \sff{Shop}_s ;
            \\
            \tLet (m,m') = \tNew \tIn \tSpawn \sff{Mother}_m ;
            \\
            \sff{Son}_{s',m'}
        \end{array}
    }
\]

Note that the son cannot be spawned, as he returns the result of the computation (the book).
In \Cref{s:LASTn:LAST,s:LASTn:LASTn} we will discuss the behavior and typing of these terms.

\smallskip
After presenting \APCP next, in \Cref{s:LASTn:LASTn} we formally define \LASTn and  illustrate the language and its runtime semantics: how new channels are created, new threads are spawned, and messages are written to and read from buffers.

\section{\APCP: Asynchronous Priority-based Classical Processes}
\label{s:APCP}

In this section, we define \APCP, a session-typed $\pi$-calculus in which processes communicate asynchronously  on connected channel endpoints.
We further refer to endpoints by the names that represent them.
As already discussed, the output of messages (names and labels) is non-blocking, and explicit continuations ensure the ordering of messages within a session.
In our type system, names are assigned types that specify two-party protocols, in the style of binary session types~\cite{conf/concur/Honda93}, following the Curry-Howard correspondences between linear logic and session types~\cite{conf/concur/CairesP10,conf/icfp/Wadler12}.

\APCP combines the salient features of Dardha and Gay's \PCP~\cite{conf/fossacs/DardhaG18} with DeYoung \etal's semantics for asynchronous communication~\cite{conf/csl/DeYoungCPT12}.
Recursion---not present in the works by Dardha and Gay and DeYoung \etal---is an orthogonal feature, with syntax inspired by the work of Toninho \etal~\cite{conf/tgc/ToninhoCP14}.

As in \PCP, types in \APCP rely on \emph{priority} annotations, which enable cyclic connections while ruling out circular dependencies between sessions.
A key insight of our work is that asynchrony  significantly improves priority management:
while \PCP's outputs are blocking and thus require priority checks, the non-blocking outputs of \APCP do not need priority checks.
This effectively reduces priority checks by half (see \Cref{r:APCP:simplerPrios}).

Properties of well-typed \APCP processes are \emph{type preservation} (\Cref{t:APCP:sr}) and \emph{deadlock-freedom} (\Cref{t:APCP:df}).
This includes cyclically connected processes, which priority-annotated types guarantee free from circular dependencies that may cause deadlock.

\subsection{The Process Language}
\label{s:APCP:APCP:proclang}

\ifthesis\else
\begin{figure}[t!]
    Process syntax:
    \begin{align*}
        P, Q &
        \begin{array}[t]{@{}l@{}lr@{\kern2em}l@{\kern1em}lr@{}}
            {} ::= {} &
            \pOut x[a,b] & \text{send}
            & \sepr &
            \pIn x(y,z) ; P & \text{receive}
            \\ \sepr* &
            \pSel x[b] < \ell & \text{selection}
            & \sepr &
            \pBra x(z) > {\{ i : P \}_{i \in I}} & \text{branch}
            \\ \sepr* &
            \pRes{xy} P & \text{restriction}
            & \sepr &
            P \| Q & \text{parallel}
            \\ \sepr* &
            \0 & \text{inaction}
            & \sepr &
            \pFwd [x<>y] & \text{forwarder}
            \\ \sepr* &
            \pRec X(\tilde{z}) ; P & \text{recursive definition}
            & \sepr &
            \pCall X<\tilde{z}> & \text{recursive call}
        \end{array}
    \end{align*}

    \dashes

    Structural congruence:
    \begin{mathpar}
        \begin{bussproof}[cong-alpha]
            \bussAssume{
                P \equiv_\alpha Q
            }
            \bussUn{
                P \equiv Q
            }
        \end{bussproof}
        \and
        \begin{bussproof}[cong-par-unit]
            \bussAx{
                P \| \0 \equiv P
            }
        \end{bussproof}
        \and
        \begin{bussproof}[cong-par-comm]
            \bussAx{
                P \| Q \equiv Q \| P
            }
        \end{bussproof}
        \and
        \begin{bussproof}[cong-par-assoc]
            \bussAx{
                P \| (Q \| R) \equiv (P \| Q) \| R
            }
        \end{bussproof}
        \and
        \begin{bussproof}[cong-scope]
            \bussAssume{
                x,y \notin \fn(P)
            }
            \bussUn{
                P \| \pRes{xy} Q \equiv \pRes{xy} ( P \| Q )
            }
        \end{bussproof}
        \and
        \begin{bussproof}[cong-res-comm]
            \bussAx{
                \pRes{xy} \pRes{zw} P \equiv \pRes{zw} \pRes{xy} P
            }
        \end{bussproof}
        \and
        \begin{bussproof}[cong-res-symm]
            \bussAx{
                \pRes{xy} P \equiv \pRes{yx} P
            }
        \end{bussproof}
        \and
        \begin{bussproof}[cong-res-inact]
            \bussAx{
                \pRes{xy} \0 \equiv \0
            }
        \end{bussproof}
        \and
        \begin{bussproof}[cong-fwd-symm]
            \bussAx{
                \pFwd [x<>y] \equiv \pFwd [y<>x]
            }
        \end{bussproof}
        \and
        \begin{bussproof}[cong-res-fwd]
            \bussAx{
                \pRes{xy} \pFwd [x<>y] \equiv \0
            }
        \end{bussproof}
        \and
        \begin{bussproof}[cong-unfold]
            \bussAx{
                \pRec X(x_1,\ldots,x_n) ; P \equiv P \big\{ \big( \pRec X(y_1,\ldots,y_n) ; P \{ y_1/x_1,\ldots,y_n/x_n \} \big) / \pCall X<y_1,\ldots,y_n> \big\}
            }
        \end{bussproof}
    \end{mathpar}

    \dashes

    Reduction:
    \begin{mathpar}
        \begin{bussproof}[red-send-recv]
            \bussAx{
                \pRes{xy} ( \pOut x[a,b] \| \pIn y(z,y') ; Q )
                \redd
                Q \{a/z,b/y'\}
            }
        \end{bussproof}
        \and
        \begin{bussproof}[red-sel-bra]
            \bussAssume{
                j \in I
            }
            \bussUn{
                \pRes{xy} ( \pSel x[b] < j \| \pBra y(y') > {\{ i : Q_i \}_{i \in I}} )
                \redd
                Q_j \{b/y'\}
            }
        \end{bussproof}
        \and
        \begin{bussproof}[red-fwd]
            \bussAssume{
                y \neq z
            }
            \bussUn{
                \pRes{xy}( \pFwd [x<>z] \| P )
                \redd
                P \{z/y\}
            }
        \end{bussproof}
        \and
        \begin{bussproof}[red-cong]
            \bussAssume{
                P \equiv P'
            }
            \bussAssume{
                P' \redd Q'
            }
            \bussAssume{
                Q' \equiv Q
            }
            \bussTern{
                P \redd Q
            }
        \end{bussproof}
        \and
        \begin{bussproof}[red-res]
            \bussAssume{
                P \redd Q
            }
            \bussUn{
                \pRes{xy} P \redd \pRes{xy} Q
            }
        \end{bussproof}
        \and
        \begin{bussproof}[red-par]
            \bussAssume{
                P \redd Q
            }
            \bussUn{
                P \| R \redd Q \| R
            }
        \end{bussproof}
    \end{mathpar}

    \caption{Definition of \APCP's process language.}
    \label{f:procdef}
\end{figure}
\fi

We write $a,b,c,\ldots,x, y, z, \ldots$ to denote (channel) \emph{names} (also known as \emph{names}); by convention we use the early letters of the alphabet for the objects of output-like prefixes.
Also, we write $\tilde{x}, \tilde{y}, \tilde{z}, \ldots$ to denote sequences of names.
In \APCP,  communication is
\emph{asynchronous} (cf.~\cite{conf/ecoop/HondaT91,conf/occ/HondaT91,report/Boudol92}) and
\emph{dyadic}: each communication involves the transmission of a pair of names, a message name and a continuation name.
With a slight abuse of notation, we sometimes write $x_i \in  \tilde{x}$ to refer to a specific element in the sequence $\tilde{x}$.
Also, we write $i, j, k, \ldots$ to denote \emph{labels} for choices and $I, J, K, \ldots$ to denote sets of labels.
We write \ifLASTn\else $X, Y, \ldots$ to denote \emph{recursion variables}, and \fi $P,Q, \ldots$ to denote processes.

\ifthesis
\begin{definition}[\APCP Syntax]
    \ifAPCP \label{d:APCP:APCP:syntax}
    \else\ifLASTn \label{d:LASTn:APCP:syntax}
    \else\ifmpstAPCP \label{d:mpstAPCP:APCP:syntax} \fi\fi\fi
    The syntax of \APCP processes is as follows:
    \ifLASTn 
    \begin{align*}
        P, Q &
        \begin{array}[t]{@{}l@{}lr@{\kern2em}l@{\kern1em}lr@{}}
            {} ::= {} &
            \pOut x[a,b] & \text{send}
            & \sepr &
            \pIn x(y,z) ; P & \text{receive}
            \\ \sepr* &
            \pSel x[b] < \ell & \text{selection}
            & \sepr &
            \pBra x(z) > {\{ i : P \}_{i \in I}} & \text{branch}
            \\ \sepr* &
            \pRes{xy} P & \text{restriction}
            & \sepr &
            P \| Q & \text{parallel}
            \\ \sepr* &
            \0 & \text{inaction}
            & \sepr &
            \pFwd [x<>y] & \text{forwarder}
        \end{array}
    \end{align*}
    \else
    \begin{align*}
        P, Q &
        \begin{array}[t]{@{}l@{}lr@{\kern2em}l@{\kern1em}lr@{}}
            {} ::= {} &
            \pOut x[a,b] & \text{send}
            & \sepr &
            \pIn x(y,z) ; P & \text{receive}
            \\ \sepr* &
            \pSel x[b] < \ell & \text{selection}
            & \sepr &
            \pBra x(z) > {\{ i : P \}_{i \in I}} & \text{branch}
            \\ \sepr* &
            \pRes{xy} P & \text{restriction}
            & \sepr &
            P \| Q & \text{parallel}
            \\ \sepr* &
            \0 & \text{inaction}
            & \sepr &
            \pFwd [x<>y] & \text{forwarder}
            \\ \sepr* &
            \pRec X(\tilde{z}) ; P & \text{recursive definition}
            & \sepr &
            \pCall X<\tilde{z}> & \text{recursive call}
        \end{array}
    \end{align*}
    \fi
\end{definition}
\fi

\ifthesis
\noindent
\else
\Cref{f:procdef} (top) gives the syntax of processes.
\fi
The send $\pOut x[a,b]$ sends along $x$ a message name $a$ and a continuation name $b$.
The receive~$\pIn x(y,z) ; P$ blocks until on $x$ a message and continuation name are received (referred to in~$P$ as the placeholders $y$ and $z$, respectively), binding $y$ and $z$ in $P$.
The selection~\mbox{$\pSel x[b] < i$} sends along $x$ a label $i$ and a continuation name $b$.
The branch $\pBra x(z) > {\{ i : P_i \}_{i \in I}}$ blocks until it receives on $x$ a label $i \in I$ and a continuation name (referred to in $P_i$ as the placeholder~$z$), binding $z$ in each $P_i$.
In the rest of this \ifthesis chapter\else paper\fi, we refer to sends, receives, selections, and branches---including their continuations, if any---as \emph{prefixes}.
We refer to sends and selections collectively as \emph{outputs}, and to receives and branches as \emph{inputs}.

Restriction $\pRes{xy} P$ binds $x$ and $y$ in $P$, thus declaring them as the two names of a channel and enabling communication, as in~\cite{journal/ic/Vasconcelos12}.
The process $P \| Q$ denotes the parallel composition of $P$ and $Q$.
The process $\0$ denotes inaction.
The forwarder $\pFwd [x<>y]$ is a primitive copycat process that links together $x$ and $y$.
We say a forwarder $\pFwd [x<>y]$ in $P$ is \emph{independent} if $P$ does not bind $x$ and $y$ together through restriction (and \emph{dependent} if it does).
\ifLASTn\else
The process $\pRec X(\tilde{z}) ; P$ denotes a recursive definition, binding occurrences of $X$ in $P$; the names $\tilde{z}$ form a context for $P$.
Then $P$ may contain recursive calls $\pCall X<\tilde{z}>$ that indicate a repitition of $P$, providing the names $\tilde{z}$ as context.
We only consider contractive recursion, disallowing processes with subexpressions of the form $\pRec X_1(\tilde{z}) ; \ldots ; \pRec X_n(\tilde{z}) ; \pCall X_1<\tilde{z}>$.
\fi

Names
\ifLASTn\else
and recursion variables
\fi
are free unless otherwise stated (i.e., unless they are bound somehow).
We write $\fn(P)$
\ifLASTn\else
and $\frv(P)$
\fi
for the sets of free names
\ifLASTn\else
and free recursion variables
\fi
of $P$\ifLASTn, \else, respectively, \fi
and $\bn(P)$ for the set of bound names of $P$.
Also, we write $P \{ x/y \}$ to denote the capture-avoiding substitution of the free occurrences of $y$ in $P$ for $x$.
\ifLASTn\else
Notation $P \big\{ \big( \pRec X(y_1,\ldots,y_n) ; P' \big) / \pCall X<y_1,\ldots,y_n> \big\}$ denotes the substitution of occurrences of recursive calls $\pCall X<y_1,\ldots,y_n>$ in $P$ with the recursive definition $\pRec X(y_1,\ldots,y_n) ; P'$, which we call \emph{unfolding} recursion.
\fi
We write sequences of substitutions $P \{ x_1/y_1 \} \ldots \{ x_n/y_n \}$ as $P \{ x_1/y_1, \ldots, x_n/y_n \}$.

\ifLASTn\else
\ifmpstAPCP\else
Except for asynchrony and recursion, there are minor differences with respect to the languages of Dardha and Gay~\cite{conf/fossacs/DardhaG18} and DeYoung \etal~\cite{conf/csl/DeYoungCPT12}.
Unlike Dardha and Gay's, our syntax does not include empty send and receive prefixes that explicitly close channels; this simplifies the type system.
We also do not include the operator for replicated servers, denoted $\pSrv x(y) ; P$, which is present in~\cite{conf/fossacs/DardhaG18,conf/csl/DeYoungCPT12}.
Although replication can be handled without difficulties, we omit it here; we prefer focusing on recursion, because it fits well with the examples we consider.
See \Cref{s:APCP:APCP:exts} for further discussion.

\paragraph{A convenient notation.}

In a send $\pOut x[a,b]$, both $a$ and $b$ are free; they can be bound to a continuation process using parallel composition and restriction, as in
\[
    \pRes{ay} \pRes{bz} ( \pOut x[a,b] \| P_{y,z} ).
\]
The same applies to selection $\pSel x[b] < i$.
We introduce useful notations that elide the restrictions and continuation names:
\fi
\fi

\begin{notation}[Derivable Bound Communication]
    \ifLASTn \label{n:LASTn:APCP:sugar}
    \else\ifmpstAPCP \label{n:mpstAPCP:APCP:sugar}
    \else \label{n:APCP:sugar} \fi\fi
    We use the following syntactic sugar:
    \begin{align*}
        \pOut* x[y] \cdot P
        & := \pRes{ya} \pRes{zb} ( \pOut x[a,b] \| P \{ z/x \} )
        &
        \pSel* x < \ell \cdot P
        & := \pRes{zb} ( \pSel x[b] < \ell \| P \{ z/x \} )
        \\
        \pIn x(y) ; P
        & := \pIn x(y,z) ; P \{ z/x \}
        &
        \pBra* x > {\{ i : P_i \}_{i \in I}}
        & := \pBra x(z) > {\{ i : P_i \{ z/x \} \}_{i \in I}}
    \end{align*}
\end{notation}

\noindent
Note the use of `$\,\cdot\,$' instead of `$\,;\,$' in sending and selection to stress that they are non-blocking.

\paragraph{Operational semantics.}

We define a reduction relation for processes ($P \redd Q$) that formalizes how complementary outputs/inputs on connected names may synchronize.
As usual for $\pi$-calculi, reduction relies on \emph{structural congruence} ($P \equiv Q$), which relates processes with minor syntactic differences.
\ifthesis\else
Structural congruence is the smallest congruence on the syntax of processes (\Cref{f:procdef} (top)) satisfying the axioms in \Cref{f:procdef} (center).
\fi

\ifthesis
\begin{definition}[Structural Congruence ($\equiv$) for \APCP]
    \ifLASTn \label{d:LASTn:APCP:sc}
    \else\ifmpstAPCP \label{d:mpstAPCP:APCP:sc}
    \else \label{d:APCP:APCP:strcong} \fi \fi
    Structural congruence for \APCP, denoted $P \equiv Q$, is the smallest congruence on the syntax of processes \ifLASTn \else\ifmpstAPCP \else (\Cref{d:APCP:APCP:syntax}) \fi\fi satisfying the following axioms:
    \begin{mathpar}
        \begin{bussproof}[cong-alpha]
            \bussAssume{
                P \equiv_\alpha Q
            }
            \bussUn{
                P \equiv Q
            }
        \end{bussproof}
        \and
        \begin{bussproof}[cong-par-unit]
            \bussAx{
                P \| \0 \equiv P
            }
        \end{bussproof}
        \and
        \begin{bussproof}[cong-par-comm]
            \bussAx{
                P \| Q \equiv Q \| P
            }
        \end{bussproof}
        \and
        \begin{bussproof}[cong-par-assoc]
            \bussAx{
                P \| (Q \| R) \equiv (P \| Q) \| R
            }
        \end{bussproof}
        \and
        \begin{bussproof}[cong-scope]
            \bussAssume{
                x,y \notin \fn(P)
            }
            \bussUn{
                P \| \pRes{xy} Q \equiv \pRes{xy} ( P \| Q )
            }
        \end{bussproof}
        \and
        \begin{bussproof}[cong-res-comm]
            \bussAx{
                \pRes{xy} \pRes{zw} P \equiv \pRes{zw} \pRes{xy} P
            }
        \end{bussproof}
        \and
        \begin{bussproof}[cong-res-symm]
            \bussAx{
                \pRes{xy} P \equiv \pRes{yx} P
            }
        \end{bussproof}
        \and
        \begin{bussproof}[cong-res-inact]
            \bussAx{
                \pRes{xy} \0 \equiv \0
            }
        \end{bussproof}
        \and
        \begin{bussproof}[cong-fwd-symm]
            \bussAx{
                \pFwd [x<>y] \equiv \pFwd [y<>x]
            }
        \end{bussproof}
        \and
        \begin{bussproof}[cong-res-fwd]
            \bussAx{
                \pRes{xy} \pFwd [x<>y] \equiv \0
            }
        \end{bussproof}
        \ifLASTn\else
        \and
        \begin{bussproof}[cong-unfold]
            \bussAx{
                \pRec X(x_1,\ldots,x_n) ; P \equiv P \big\{ \big( \pRec X(y_1,\ldots,y_n) ; P \{ y_1/x_1,\ldots,y_n/x_n \} \big) / \pCall X<y_1,\ldots,y_n> \big\}
            }
        \end{bussproof}
        \fi
    \end{mathpar}
\end{definition}
\fi

\ifthesis\noindent\fi
Structural congruence defines the following properties for processes.
Processes are equivalent up to $\alpha$-equivalence (Rule~\ruleLabel{cong-alpha}).
Parallel composition is associative (Rule~\ruleLabel{cong-par-assoc}) and commutative (Rule~\ruleLabel{cong-par-comm}), with unit $\0$ (Rule~\ruleLabel{cong-par-unit}).
A parallel process may be moved into or out of a restriction as long as the bound channels do not occur free in the moved process (Rule~\ruleLabel{cong-scope}): this is \emph{scope inclusion} and \emph{scope extrusion}, respectively.
Restrictions on inactive processes may be dropped (Rule~\ruleLabel{cong-res-inact}), and the order of names in restrictions and of consecutive restrictions does not matter (Rules~\ruleLabel{cong-res-symm} and~\ruleLabel{cong-res-comm}, respectively).
Forwarders are symmetric (Rule~\ruleLabel{cong-fwd-symm}), and equivalent to inaction if both names are bound together through restriction (Rule~\ruleLabel{cong-res-fwd}).
\ifLASTn\else
Finally, a recursive definition is equivalent to its unfolding (Rule~\ruleLabel{cong-unfold}), replacing any recursive calls with copies of the recursive definition, where the recursive definition's contextual names are pairwise substituted for by the call's names.
\fi
\ifLASTn\else

\fi
As we will see next, the semantics of \APCP is closed under structural congruence.
\ifLASTn\else
This means that processes are \emph{equi-recursive}; however, \APCP's typing discipline (described \ifmpstAPCP next\else \Cref{s:APCP:APCP:types}\fi) treats recursive types as \emph{iso-recursive} (see, e.g., Pierce~\cite{book/Pierce02}).
\fi

\ifthesis
We now define the reduction relation $P \redd Q$.

\begin{definition}[Reduction ($\protect\redd$) for \APCP]
    \ifLASTn \label{d:LASTn:APCP:redd}
    \else\ifmpstAPCP \label{d:mpstAPCP:APCP:redd}
    \else \label{d:APCP:APCP:redd} \fi\fi
    Reduction for \APCP is a relation between processes, denoted $P \redd Q$.
    It is defined by the following rules:
    \begin{mathpar}
        \begin{bussproof}[red-send-recv]
            \bussAx{
                \pRes{xy} ( \pOut x[a,b] \| \pIn y(z,y') ; Q )
                \redd
                Q \{a/z,b/y'\}
            }
        \end{bussproof}
        \and
        \begin{bussproof}[red-sel-bra]
            \bussAssume{
                j \in I
            }
            \bussUn{
                \pRes{xy} ( \pSel x[b] < j \| \pBra y(y') > {\{ i : Q_i \}_{i \in I}} )
                \redd
                Q_j \{b/y'\}
            }
        \end{bussproof}
        \and
        \begin{bussproof}[red-fwd]
            \bussAssume{
                x,y \neq z
            }
            \bussUn{
                \pRes{xy}( \pFwd [x<>z] \| P )
                \redd
                P \{z/y\}
            }
        \end{bussproof}
        \and
        \begin{bussproof}[red-cong]
            \bussAssume{
                P \equiv P'
            }
            \bussAssume{
                P' \redd Q'
            }
            \bussAssume{
                Q' \equiv Q
            }
            \bussTern{
                P \redd Q
            }
        \end{bussproof}
        \and
        \begin{bussproof}[red-res]
            \bussAssume{
                P \redd Q
            }
            \bussUn{
                \pRes{xy} P \redd \pRes{xy} Q
            }
        \end{bussproof}
        \and
        \begin{bussproof}[red-par]
            \bussAssume{
                P \redd Q
            }
            \bussUn{
                P \| R \redd Q \| R
            }
        \end{bussproof}
    \end{mathpar}
    We write $\redd*$ for the reflexive, transitive closure of $\redd$.
\end{definition}
\fi

\ifthesis\else
We define the reduction relation $P \redd Q$ by the axioms and closure rules in \Cref{f:procdef} (bottom).
We write $\redd^\ast$ for the reflexive, transitive closure of $\redd$.
\fi
\ifthesis\noindent\fi
Rule~$\ruleLabel{red-send-recv}$ synchronizes a send and a receive on connected names and substitutes the message and continuation names.
Rule~$\ruleLabel{red-sel-bra}$ synchronizes a selection and a branch:
the received label determines the continuation process, substituting the continuation name appropriately.
Rule~$\ruleLabel{red-fwd}$ implements the forwarder as a substitution.
Rules~$\ruleLabel{red-cong}$, $\ruleLabel{red-res}$, and~$\ruleLabel{red-par}$ close reduction under structural congruence, restriction, and parallel composition, respectively.

\NewDocumentCommand{\cyclicAsyncXReddP}{s}{%
    \[
        P \IfBooleanTF{#1}{=}{\deq} \pRes{zu} \begin{array}[t]{@{}l@{}}
            \big(
                \pRes{xy} \begin{array}[t]{@{}l@{}}
                    \big(
                        \pRes{ax'} ( \pOut x[v_1,a] \| \pOut x'[v_2,b] )
                        \\
                        {} \| \pRes{cz'} ( \pOut z[v_3,c] \| \pIn y(w_1,y') ; \pIn y'(w_2,y'') ; Q )
                    \big)
                \end{array}
                \\
                {} \| \pIn u(w_3,u') ; R
            \big)
        \end{array}
    \]
}

Notice how sends and selections output free names.
This is different from the works by Dardha and Gay~\cite{conf/fossacs/DardhaG18} and DeYoung \etal~\cite{conf/csl/DeYoungCPT12}, where, following an internal mobility discipline~\cite{journal/tcs/Boreale98}, communication involves bound names only.
As we show in the next subsection, this kind of \emph{bound output} is derivable in \APCP (cf.\ \Cref{n:APCP:sugar} and \Cref{t:APCP:derivable}).

\begin{example}
    \ifLASTn\label{x:LASTn:APCP:redd}\else\label{x:APCP:redd}\fi
    To illustrate the preservation of order \emph{within} a session and the asynchrony \emph{between} different sessions, we consider the following process:
    \cyclicAsyncXReddP
    The process $P$ defines two consecutive sends on a session from $x$ to $y$, and an asynchronous send on a session from $z$ to $u$.
    Two reductions are possible from $P$:
    \begin{align}
        P &\redd \pRes{zu} \big(
            \pRes{ax'} \big(
                \pOut x'[v_2,b] \| \pRes{cz'} ( \pOut z[v_3,c] \| \pIn a(w_2,y'') ; Q \{ v_1/w_1 \} )
            \big) \| \pIn u(w_3,u') ; R
        \big)
        \label{eq:APCP:xRed1}
        \\
        P &\redd \pRes{cz'} \big(
            \pRes{xy} \big(
                \pRes{ax'} ( \pOut x[v_1,a] \| \pOut x'[v_2,b] ) \| \pIn y(w_1,y') ; \pIn y'(w_2,y'') ; Q
            \big) \| R \{ v_3/w_3 , c/u' \}
        \big)
        \label{eq:APCP:xRed2}
    \end{align}
    The reduction~\eqref{eq:APCP:xRed1} entails the synchronization of the send on $x$ and the receive on $y$; afterwards, the send on $x'$ is connected to the receive on $a$ that prefixes $Q$.
    The reduction~\eqref{eq:APCP:xRed2} entails the synchronization of the send on $z$ and the receive on $u$, connecting $Q$ and $R$ on a new session between $z'$ and $c$.
    Note that from $P$ there is no reduction involving the send on $x'$, since $x'$ is connected to the continuation name of the send on $x$ and is thus not (yet) paired with a dual receive.

    Using the sugared syntax from \Cref{n:APCP:sugar}, we can write
    \[
        P = \pRes{zu} ( \pRes{xy} ( \pOut* x[v_1] \cdot \pOut* x[v_2] \cdot \0 \| \pOut* z[v_3] \cdot \pIn y(w_1) ; \pIn y(w_2) ; Q' ) \| \pIn u(w_3) ; R' )
    \]
    where $Q' \deq Q \{ y/y'' \}$ and $R' \deq R \{ u/u' \}$.

    The following diagram illustrates all the possible reduction paths from $P$; horizontal reductions concern the session between $x$ and $y$, and diagonal reductions concern the session between $z$ and $u$:
    \[
        \begin{array}{@{}l@{\kern1pt}l@{}l@{\kern1pt}l@{}l@{\kern1pt}l@{}l@{}}
            P & \redd & P_1 & \redd & P_3
            \\[-5pt]
            & \rotatebox[origin=l]{-30}{$\redd$}
            &
            & \rotatebox[origin=l]{-30}{$\redd$}
            &
            & \rotatebox[origin=l]{-30}{$\redd$}
            \\[-5pt]
            & & P_2 & \redd & P_4 & \redd & P_5
        \end{array}
    \]
    where processes $P_1$, $P_2$, $P_3$, $P_4$, and $P_5$ are as follows:
    \begin{align*}
        P_1 &\deq \pRes{zu}( \pRes{xy}( \pOut* x[v_2] \cdot \0 \| \pOut* z[v_3] \cdot \pIn y(w_2) ; Q' \{ v_1/w_1 \} ) \| \pIn u(w_3) ; R' )
        \\
        P_2 &\deq \pRes{xy}( \pOut* x[v_1] \cdot \pOut* x[v_2] \cdot \0 \| \pIn y(w_1) ; \pIn y(w_2) ; Q' ) \| R' \{ v_3/w_3 \}
        \\
        P_3 &\deq \pRes{zu}( \pOut* z[v_3] \cdot Q' \{ v_1/w_1 , v_2/w_2 \} \| \pIn u(w_3) ; R' )
        \\
        P_4 &\deq \pRes{xy}( \pOut* x[v_2] \cdot \0 \| \pIn y(w_2) ; Q' \{ v_1/w_1 \} ) \| R' \{ v_3/w_3 \}
        \\
        P_5 &\deq Q' \{ v_1/w_1 , v_2/w_2 \} \| R' \{ v_3/w_3 \}
        \tag*{\qedhere}
    \end{align*}
\end{example}

\ifLASTn\else\ifmpstAPCP\else
Presentations of reduction for session-typed $\pi$-calculi derived from the Curry-Howard interpretations of linear logic often include rules that correspond to \emph{commuting conversions} in linear logic (cf.~\cite{conf/concur/CairesP10,conf/icfp/Wadler12,conf/fossacs/DardhaG18,conf/csl/DeYoungCPT12}), which allow rewriting processes in which a way that blocking communications on free names are ``pulled out'' of restrictions.
Commuting conversions can be easily included for \APCP (cf.\ \cite{conf/ice/vdHeuvelP21}), but we do not include them here for two reasons.
First, while Dardha and Gay~\cite{conf/fossacs/DardhaG18} rely on commuting conversions to prove deadlock-freedom, as discussed in detail in \Cref{s:APCP:APCP:results} the proof of deadlock-freedom for \APCP takes a different approach that does not require commuting conversions.
Second, although considering commuting conversions would arguably lead to a stronger Curry-Howard correspondence with linear logic, the transformations they induce on processes appear rather unrealistic from a message-passing perspective. 
\fi\fi


\subsection{The Type System}
\label{s:APCP:APCP:types}

\APCP types processes by assigning binary session types to names.
\ifLASTn\else
Following Curry-Howard interpretations, we present session types as linear logic propositions (cf., e.g., Caires \etal~\cite{journal/mscs/CairesPT16}, Wadler~\cite{conf/icfp/Wadler12}, Caires and P\'erez~\cite{conf/esop/CairesP17}, and Dardha and Gay~\cite{conf/fossacs/DardhaG18}).
We extend these propositions with recursion and \emph{priority} annotations on connectives.
Intuitively, prefixes typed with lower priority should not be blocked by those with higher priority.

\fi
We write $\pri, \pi, \rho, \ldots$ to denote priorities, and $\omega$ to denote the ultimate priority that is greater than all other priorities  and cannot be increased further.
That is, for every $\pri \in \mbb{N}$, we have $\omega > \pri$ and $\omega + \pri = \omega$.

\begin{definition}[Session Types for \APCP]
    \ifLASTn \label{d:LASTn:APCP:types}
    \else\ifmpstAPCP \label{d:mpstAPCP:APCP:types}
    \else \label{d:APCP:types} \fi\fi
    The following grammar defines the syntax of \emph{session types} $A,B$.
    Let $\pri \in \mbb{N}$.
    \begin{align*}
        A,B
        &::= A \tensor^\pri B \sepr A \parr^\pri B \sepr \oplus^\pri \{ i : A \}_{i \in I} \sepr \&^\pri \{ i : A \}_{i \in I} \sepr \bullet \ifLASTn\else \sepr \mu X . A \sepr X \fi
    \end{align*}
\end{definition}

\noindent
A name of type $A \tensor^\pri B$ (resp.\ $A \parr^\pri B$) first sends (resp.\ receives) a name of type $A$ and then behaves as $B$.
A name of type $\oplus^\pri \{ i : A_i \}_{i \in I}$ selects a label $i \in I$ and then behaves as $A_i$.
A name of type $\&^\pri \{ i : A_i \}_{i \in I}$ offers a choice: after receiving a label $i \in I$, the name behaves as $A_i$.
A name of type $\bullet$ is closed; it does not require a priority, as closed names do not exhibit behavior and thus are non-blocking.

\ifLASTn\else\ifmpstAPCP\else
Unlike Caires and Pfenning~\cite{conf/concur/CairesP10} and Dardha and Gay \cite{conf/fossacs/DardhaG18}, \APCP does not associate any behavior with closed sessions (i.e., no closing synchronizations).
Moreover, as we will see, \APCP's type system allows arbitrary parallel composition.
Atkety \etal~\cite{chapter/AtkeyLM16} have shown that, in presence of arbitrary parallel composition (i.e., linear logic's Rule~\ruleLabel{mix} that combines sequents arbitrarily, usually interpreted as session typing rule for arbitrary parallel composition), the dual propositions $\1$ and $\bot$ (usually associated with complementary communications for closing sessions) are equivalent.
Hence, since in \APCP session closing is silent, we follow Caires~\cite{report/Caires14} in conflating the types $\1$ and $\bot$ to the single, self-dual type $\bullet$ for closed sessions.
\fi

Type $\mu X . A$ denotes a recursive type, in which $A$ may contain occurrences of the recursion variable $X$.
As customary, $\mu$ is a binder: it induces the standard notions of $\alpha$-equivalence, substitution (denoted $A \{ B/X \}$), and free recursion variables (denoted $\frv(A)$).
We work with tail-recursive, contractive types, disallowing types of the form $\mu X_1 \ldots \mu X_n . X_1$ and $\mu X . X \tensor^\pri A$.
Recursive types are treated iso-recursively: there will be an explicit typing rule that unfolds recursive types, and recursive types are not equal to their unfolding.
We postpone formalizing the unfolding of recursive types, as it requires additional definitions to ensure consistency of priorities upon unfolding.

\emph{Duality}, the cornerstone notion of session types and linear logic, ensures that the two names of a channel have matching prefixes.
Furthermore, dual types must have matching priority annotations.
The following inductive definition of duality suffices for our tail-recursive types (cf.\ Gay \etal~\cite{conf/places/GayTV20}).
\fi

\newpage
\begin{definition}[Duality]
    \ifLASTn \label{d:LASTn:APCP:duality}
    \else\ifmpstAPCP \label{d:mpstAPCP:APCP:dual}
    \else \label{d:APCP:duality} \fi\fi
    The \emph{dual} of session type $A$, denoted $\ol{A}$, is defined inductively as follows:
    \ifLASTn
    \begin{align*}
        \ol{A \tensor^\pri B}
        &\deq
        \ol{A} \parr^\pri \ol{B}
        &
        \ol{\oplus^\pri \{ i : A_i \}_{i \in I}}
        &\deq
        \&^\pri \{ i : \ol{A_i} \}_{i \in I}
        &
        \ol{\bullet}
        &\deq
        \bullet
        \\
        \ol{A \parr^\pri B}
        &\deq
        \ol{A} \tensor^\pri \ol{B}
        &
        \ol{\&^\pri \{ i : A_i \}_{i \in I}}
        &\deq
        \oplus^\pri \{ i : \ol{A_i} \}_{i \in I}
    \end{align*}
    \else
    \begin{align*}
        \ol{A \tensor^\pri B}
        &\deq
        \ol{A} \parr^\pri \ol{B}
        &
        \ol{\oplus^\pri \{ i : A_i \}_{i \in I}}
        &\deq
        \&^\pri \{ i : \ol{A_i} \}_{i \in I}
        &
        \ol{\bullet}
        &\deq
        \bullet
        &
        \ol{\mu X . A}
        &\deq
        \mu X . \ol{A}
        \\
        \ol{A \parr^\pri B}
        &\deq
        \ol{A} \tensor^\pri \ol{B}
        &
        \ol{\&^\pri \{ i : A_i \}_{i \in I}}
        &\deq
        \oplus^\pri \{ i : \ol{A_i} \}_{i \in I}
        &
        &
        &
        \ol{X}
        &\deq
        X
    \end{align*}
    \fi
\end{definition}

The priority of a type is determined by the priority of the type's outermost connective:

\begin{definition}[Priorities]
    \ifLASTn \label{d:LASTn:APCP:prio}
    \else\ifmpstAPCP \label{d:mpstAPCP:APCP:prio}
    \else \label{d:APCP:priority} \fi\fi
    For session type $A$, $\pr(A)$ denotes its \emph{priority}:
    \ifLASTn
    \begin{align*}
        \pr(A \tensor^\pri B)
        \deq
        \pr(A \parr^\pri B)
        &\deq
        \pri
        &
        \pr(\bullet)
        &\deq
        \omega
        \\
        \pr(\oplus^\pri \{ i : A_i \}_{i \in I})
        \deq
        \pr(\&^\pri \{ i : A_i \}_{i \in I})
        &\deq
        \pri
    \end{align*}
    \else
    \begin{align*}
        \pr(A \tensor^\pri B)
        \deq
        \pr(A \parr^\pri B)
        &\deq
        \pri
        &
        \pr(\mu X . A)
        &\deq
        \pr(A)
        \\
        \pr(\oplus^\pri \{ i : A_i \}_{i \in I})
        \deq
        \pr(\&^\pri \{ i : A_i \}_{i \in I})
        &\deq
        \pri
        &
        \pr(\bullet)
        \deq
        \pr(X)
        &\deq
        \omega
    \end{align*}
    \fi
\end{definition}

\ifLASTn\else\ifmpstAPCP\else
\noindent
The priority of $\bullet$ and $X$ is $\omega$: they denote the ``final'', non-blocking part of protocols.
Although $\tensor$ and $\oplus$ also denote non-blocking prefixes, they do block their continuation until they are received.
Hence, their priority is not constant.
\fi

We now turn to formalizing the unfolding of recursive types.
Recall the intuition that prefixes typed with lower priority should not be blocked by those with higher priority.
Based on this rationale, we observe that the unfolding of the recursive type $\mu X . A$ should not result in  $A \{ ( \mu X . A )/X \}$, as usual, but that the priorities of the unfolded type should be \emph{increased}.

\ifmpstAPCP\else
\begin{example}
    \label{x:APCP:equiRec}
    Consider the recursive type $\mu X . A \parr^0 X$.
    If we unfold this type without increasing the priority, we would obtain $A \parr^0 (\mu X . A \parr^0 X)$, a type in which the priorities no longer determine a global ordering between the two receives.
    By increasing the priority in the unfolded type as in, e.g., $A \parr^0 (\mu X . A \parr^1 X)$, a global ordering is preserved.
\end{example}

We make this intuition precise by defining the \emph{lift} of priorities in types:
\fi

\begin{definition}[Lift]
    \ifLASTn \label{d:LASTn:APCP:lift}
    \else\ifmpstAPCP \label{d:mpstAPCP:APCP:lift}
    \else \label{d:APCP:lift} \fi\fi
    For proposition $A$ and $t \in \mbb{N}$, we define $\lift{t}A$ as the \emph{lift} operation:
    \ifLASTn
    \begin{align*}
        \lift{t}(A \tensor^\pri B)
        &\deq
        (\lift{t}A) \tensor^{\pri+t} (\lift{t}B)
        &
        \lift{t}(\oplus^\pri \{ i : A_i \}_{i \in I})
        &\deq
        \oplus^{\pri+t} \{ i : \lift{t}A_i \}_{i \in I}
        &
        \lift{t}\bullet
        &\deq
        \bullet
        \\
        \lift{t}(A \parr^\pri B)
        &\deq
        (\lift{t}A) \parr^{\pri+t} (\lift{t}B)
        &
        \lift{t}(\&^\pri \{ i : A_i \}_{i \in I})
        &\deq
        \&^{\pri+t} \{ i : \lift{t}A_i \}_{i \in I}
    \end{align*}
    \else
    \begin{align*}
        \lift{t}(A \tensor^\pri B)
        &\deq
        (\lift{t}A) \tensor^{\pri+t} (\lift{t}B)
        &
        \lift{t}(\oplus^\pri \{ i : A_i \}_{i \in I})
        &\deq
        \oplus^{\pri+t} \{ i : \lift{t}A_i \}_{i \in I}
        &
        \lift{t}\bullet
        &\deq
        \bullet
        \\
        \lift{t}(A \parr^\pri B)
        &\deq
        (\lift{t}A) \parr^{\pri+t} (\lift{t}B)
        &
        \lift{t}(\&^\pri \{ i : A_i \}_{i \in I})
        &\deq
        \&^{\pri+t} \{ i : \lift{t}A_i \}_{i \in I}
        \\
        \lift{t}(\mu X . A)
        &\deq
        \mu X . (\lift{t}A)
        &
        \lift{t}X
        &\deq
        X
    \end{align*}
    \fi
\end{definition}

\noindent
Henceforth, the unfolding of  $\mu X . A$ is $A \{ \big(\mu X . (\lift{t} A) \big)/X \}$, denoted $\unfold^t(\mu X . A)$, where $t \in \mbb{N}$ depends on the highest priority of the types occurring in a typing context.
We recall that we do not consider types to be equi-recursive: recursive types are not equal to their unfolding.
Recursive types can only be unfolded by typing rules, discussed next.

We now define the top priority of a type, i.e., the highest priority appearing in a type:

\begin{definition}[Top Priority]
    \ifLASTn \label{d:LASTn:APCP:top}
    \else\ifmpstAPCP \label{d:mpstAPCP:APCP:top}
    \else \label{d:APCP:top} \fi\fi
    For session type $A$, $\prtop(A)$ denotes its \emph{top priority}:
    \ifLASTn
    \begin{align*}
        \prtop(A \tensor^\pri B)
        \deq
        \prtop(A \parr^\pri B)
        &\deq
        \max(\prtop(A), \prtop(B), \pri)
        \\
        \prtop(\oplus^\pri \{ i : A_i \}_{i \in I})
        \deq
        \prtop(\&^\pri \{ i : A_i \}_{i \in I})
        &\deq
        \max(\max_{i \in I}(\prtop(A_i)), \pri)
        &
        \prtop(\bullet)
        &\deq
        0
    \end{align*}
    \else
    \begin{align*}
        \prtop(A \tensor^\pri B)
        \deq
        \prtop(A \parr^\pri B)
        &\deq
        \max(\prtop(A), \prtop(B), \pri)
        \\
        \prtop(\oplus^\pri \{ i : A_i \}_{i \in I})
        \deq
        \prtop(\&^\pri \{ i : A_i \}_{i \in I})
        &\deq
        \max(\max_{i \in I}(\prtop(A_i)), \pri)
        \\
        \prtop(\mu X . A)
        &\deq
        \prtop(A)
        &
        \prtop(\bullet)
        \deq
        \prtop(X)
        &\deq
        0
    \end{align*}
    \fi
\end{definition}

\ifmpstAPCP\else
\noindent
Notice how the top priority of $\bullet$ and $X$ is 0, in contrast to their priority (as given by \Cref{d:APCP:priority}): they do not contribute to the increase in priority for unfolding recursive types.

\paragraph{Typing rules.}

\fi\fi
\ifmpstAPCP\else
The typing rules of \APCP ensure that communications with lower priority are not blocked by those with higher priority (cf.\ Dardha and Gay~\cite{conf/fossacs/DardhaG18}).
To this end, they enforce the following laws:
\begin{enumerate}

    \item\ifLASTn\label{i:LASTn:law:sendSel}\else\ifmpstAPCP\label{i:mpstAPCP:law:sendSel}\else\label{i:APCP:law:sendSel}\fi\fi
        Sends and selections with priority $\pri$ must have continuations/payloads with priority strictly larger than $\pri$;
        
    \item\ifLASTn\label{i:LASTn:law:recvBra}\else\ifmpstAPCP\label{i:mpstAPCP:law:recvBra}\else\label{i:APCP:law:recvBra}\fi\fi
        A communication with priority $\pri$ must be prefixed only by receives and branches with priority strictly smaller than $\pri$;

    \item\ifLASTn\label{i:LASTn:law:resFwd}\else\ifmpstAPCP\label{i:mpstAPCP:law:resFwd}\else\label{i:APCP:law:resFwd}\fi\fi
        Dual communications leading to a synchronization must have equal priorities (cf.\ \ifLASTn\Cref{d:LASTn:APCP:duality}\else\ifmpstAPCP\Cref{d:mpstAPCP:APCP:duality}\else\Cref{d:APCP:duality}\fi\fi).

\end{enumerate}
\fi
Judgments are of the form \ifLASTn $\vdash P \typInf \Gamma$\else $\Omega \vdash P \typInf \Gamma$\fi, where:
\begin{itemize}

    \item
        $P$ is a process;

    \item
        $\Gamma$ is a context that assigns types to channels ($x:A$)\ifLASTn\ and allows no weakening and contraction\fi;

    \ifLASTn\else
    \item
        $\Omega$ is a context that assigns tuples of types to recursion variables ($X:(A, B, \ldots)$).
    \fi

\end{itemize}
\ifLASTn\else\ifmpstAPCP\else
A judgment $\Omega \vdash P \typInf \Gamma$ then means that $P$ can be typed in accordance with the type assignments for names recorded in $\Gamma$ and the recursion variables in $\Omega$.
Intuitively, the latter context ensures that types concur between the context names of recursive definitions and calls.
\fi
Both contexts $\Gamma$ and $\Omega$ obey \emph{exchange}: assignments may be silently reordered.
$\Gamma$ is \emph{linear}, disallowing \emph{weakening} (i.e., all assignments must be used; with the exception of names typed $\bullet$) and \emph{contraction} (i.e., assignments may not be duplicated).
$\Omega$ allows weakening and contraction, because a recursive definition may be called \emph{zero or more} times.
\fi

The empty context is written $\emptyset$.
In writing $\Gamma, x:A$ we assume that $x \notin \dom(\Gamma)$\ifLASTn. \else~(and similarly for $\Omega$). \fi
We \ifLASTn\else write~$\lift{t} \Gamma$ to denote the component-wise extension of lift (\ifLASTn\Cref{d:LASTn:APCP:lift}\else\ifmpstAPCP\Cref{d:mpstAPCP:APCP:lift}\else\Cref{d:APCP:lift}\fi\fi) to typing contexts.
Also, we \fi write $\pr(\Gamma)$ to denote the least of the priorities of all types in $\Gamma$ (\ifLASTn\Cref{d:LASTn:APCP:prio}\else\ifmpstAPCP\Cref{d:mpstAPCP:APCP:prio}\else\Cref{d:APCP:priority}\fi\fi).
An assignment $\tilde{z}:\tilde{A}$ means $z_1:A_1, \ldots, z_k:A_k$.
\ifLASTn\else We define the top priority of a sequence of types $\prtop(\tilde{A})$ as $\max_{A_i \in \tilde{A}}(\prtop(A_i))$.\fi

\begin{figure}[t!]
    \begin{mathpar}
        \begin{bussproof}[typ-send]
            \bussAssume{
                \pri < \pr(A),\pr(B)
            }
            \bussUn{
                \ifLASTn\else \Omega \fi
                \vdash \pOut x[y,z] \typInf x:A \tensor^\pri B, y:\ol{A}, z:\ol{B}
            }
        \end{bussproof}
        \and
        \begin{bussproof}[typ-recv]
            \bussAssume{
                \ifLASTn\else \Omega \fi
                \vdash P \typInf \Gamma, y:A, z:B
            }
            \bussAssume{
                \pri < \pr(\Gamma)
            }
            \bussBin{
                \ifLASTn\else \Omega \fi
                \vdash \pIn x(y, z) ; P \typInf \Gamma, x:A \parr^\pri B
            }
        \end{bussproof}
        \and
        \begin{bussproof}[typ-sel]
            \bussAssume{
                j \in I
            }
            \bussAssume{
                \pri < \pr(A_j)
            }
            \bussBin{
                \ifLASTn\else \Omega \fi
                \vdash \pSel x[z] < j \typInf x:\oplus^\pri \{ i : A_i \}_{i \in I}, z:\ol{A_j}
            }
        \end{bussproof}
        \and
        \begin{bussproof}[typ-bra]
            \bussAssume{
                \forall i \in I.~ \ifLASTn\else \Omega \fi \vdash P_i \typInf \Gamma, z:A_i
            }
            \bussAssume{
                \pri < \pr(\Gamma)
            }
            \bussBin{
                \ifLASTn\else \Omega \fi
                \vdash \pBra x(z) > {\{ i : P_i \}_{i \in I}} \typInf \Gamma, x:\&^\pri \{ i : A_i \}_{i \in I}
            }
        \end{bussproof}
        \and
        \begin{bussproof}[typ-end]
            \bussAssume{
                \ifLASTn\else \Omega \fi
                \vdash P \typInf \Gamma
            }
            \bussUn{
                \ifLASTn\else \Omega \fi
                \vdash P \typInf \Gamma, x:\bullet
            }
        \end{bussproof}
        \and
        \begin{bussproof}[typ-par]
            \bussAssume{
                \ifLASTn\else \Omega \fi
                \vdash P \typInf \Gamma
            }
            \bussAssume{
                \ifLASTn\else \Omega \fi
                \vdash Q \typInf \Delta
            }
            \bussBin{
                \ifLASTn\else \Omega \fi
                \vdash P \| Q \typInf \Gamma, \Delta
            }
        \end{bussproof}
        \and
        \begin{bussproof}[typ-res]
            \bussAssume{
                \ifLASTn\else \Omega \fi
                \vdash P \typInf \Gamma, x:A, y:\ol{A}
            }
            \bussUn{
                \ifLASTn\else \Omega \fi
                \vdash \pRes{xy} P \typInf \Gamma
            }
        \end{bussproof}
        \and
        \begin{bussproof}[typ-inact]
            \bussAx{
                \ifLASTn\else \Omega \fi
                \vdash \0 \typInf \emptyset
            }
        \end{bussproof}
        \and
        \begin{bussproof}[typ-fwd]
            \bussAx{
                \ifLASTn\else \Omega \fi
                \vdash \pFwd [x<>y] \typInf x:\ol{A}, y:A
            }
        \end{bussproof}
        \ifLASTn\else
        \and
        \begin{bussproof}[typ-rec]
            \bussAssume{
                \Omega, X:{\tilde{A}} \vdash P \typInf \tilde{z}:\tilde{U}
            }
            \bussAssume{
                t \in \mbb{N} > \prtop(\tilde{A})
            }
            \bussAssume{
                \forall U_i \in \tilde{U}.~ U_i = \unfold^t(\mu X . A_i)
            }
            \bussTern{
                \Omega \vdash \pRec X(\tilde{z}) ; P \typInf \tilde{z}:\widetilde{\mu X . A}
            }
        \end{bussproof}
        \and
        \begin{bussproof}[typ-var]
            \bussAssume{
                t \in \mbb{N}
            }
            \bussAssume{
                \forall U_i \in \tilde{U}.~ U_i = \mu X . \lift{t} A_i
            }
            \bussBin{
                \Omega, X:\tilde{A} \vdash \pCall X<\tilde{z}> \typInf \tilde{z}:\tilde{U}
            }
        \end{bussproof}
        \fi
        \\ \dashes \\
        \begin{bussproof}[typ-send$\star$]
            \bussAssume{
                \ifLASTn\else \Omega \fi
                \vdash P \typInf \Gamma, y:A, x:B
            }
            \bussAssume{
                \pri < \pr(A),\pr(B)
            }
            \bussBin{
                \ifLASTn\else \Omega \fi
                \vdash \pOut* x[y] \cdot P \typInf \Gamma, x:A \tensor^\pri B
            }
        \end{bussproof}
        \and
        \begin{bussproof}[typ-sel$\star$]
            \bussAssume{
                \ifLASTn\else \Omega \fi
                \vdash P \typInf \Gamma, x:A_j
            }
            \bussAssume{
                j \in I
            }
            \bussAssume{
                \pri < \pr(A_j)
            }
            \bussTern{
                \ifLASTn\else \Omega \fi
                \vdash \pSel* x < j \cdot P \typInf \Gamma, x:\oplus^\pri \{ i : A_i \}_{i \in I}
            }
        \end{bussproof}
        \ifLASTn\else
        \and
        \begin{bussproof}[typ-lift]
            \bussAssume{
                \ifLASTn\else \Omega \fi
                \vdash P \typInf \Gamma
            }
            \bussAssume{
                t \in \mbb{N}
            }
            \bussBin{
                \ifLASTn\else \Omega \fi
                \vdash P \typInf \lift{t} \Gamma
            }
        \end{bussproof}
        \fi
    \end{mathpar}

    \caption{The typing rules of \APCP (top) and derivable rules (bottom).}
    \ifLASTn \label{f:LASTn:APCP:typing}
    \else\ifmpstAPCP \label{f:mpstAPCP:APCP:typing}
    \else \label{f:APCP:typingRules} \fi\fi
\end{figure}

\ifLASTn\Cref{f:LASTn:APCP:typing}\else\ifmpstAPCP\Cref{f:mpstAPCP:APCP:typing}\else\Cref{f:APCP:typingRules}\fi\fi~(top) gives the typing rules.
We describe the typing rules from a \emph{bottom-up} perspective.
Rule~\ruleLabel{typ-send} types a send; this rule does not have premises to provide a continuation process, leaving the free names to be bound to a continuation process using Rules~\ruleLabel{typ-par} and~\ruleLabel{typ-res} (discussed hereafter).
Similarly, Rule~\ruleLabel{typ-sel} types an unbound selection.
Both rules require that the priority of the subject is lower than the priorities of both objects (continuation and payload)\ifmpstAPCP\else---this enforces Law~\ifLASTn\labelcref{i:LASTn:law:sendSel}\else\labelcref{i:APCP:law:sendSel}\fi\fi.
Rules~\ruleLabel{typ-recv} and~\ruleLabel{typ-bra} type receives and branches, respectively.
In both cases, the used name's priority must be lower than the priorities of the other types in the continuation's typing\ifmpstAPCP\else---this enforces Law~\ifLASTn\labelcref{i:LASTn:law:recvBra}\else\labelcref{i:APCP:law:recvBra}\fi\fi.

Rule~\ruleLabel{typ-par} types the parallel composition of two processes that do not share assignments on the same names.
Rule~\ruleLabel{typ-res} types a restriction, where the two restricted names must be of dual type and thus have matching priority\ifmpstAPCP\else---this enforces Law~\ifLASTn\labelcref{i:LASTn:law:resFwd}\else\labelcref{i:APCP:law:resFwd}\fi\fi.
Rule~\ruleLabel{typ-end} silently removes a closed name from the typing context.
Rule~\ruleLabel{typ-inact} types an inactive process with no names.
Rule~\ruleLabel{typ-fwd} types forwarding between names of dual type\ifmpstAPCP\else---this also enforces Law~\ifLASTn\labelcref{i:LASTn:law:resFwd}\else\labelcref{i:APCP:law:resFwd}\fi\fi.

\ifmpstAPCP\else
\begin{example}
    To illustrate the typing rules discussed so far, \ifLASTn consider the following process $P$\else we recall process $P$ from \Cref{x:APCP:redd}\fi:
    \cyclicAsyncXReddP*

    We give the typing of the two consecutive sends on $x$\ifLASTn\else (omitting the context $\Omega$)\fi:
    \[
        \begin{bussproof}
            \bussAssume{
                \pri < \pr(A_1),\pi
            }
            \bussUn[\ruleLabel{typ-send}]{
                \vdash \pOut x[v_1,a] \typInf \begin{array}[t]{@{}l@{}}
                    x:A_1 \tensor^\pri A_2 \tensor^\pi B,
                    \\
                    v_1:\ol{A_1}, a:\ol{A_2 \tensor^\pi B}
                \end{array}
            }
            \bussAssume{
                \pi < \pr(A_2),\pr(B)
            }
            \bussUn[\ruleLabel{typ-send}]{
                \vdash \pOut x'[v_2,b] \typInf \begin{array}[t]{@{}l@{}}
                    x':A_2 \tensor^\pi B,
                    \\
                    v_2:\ol{A_2}, b:\ol{B}
                \end{array}
            }
            \bussBin[\ruleLabel{typ-par}]{
                \vdash \pOut x[v_1,a] \| \pOut x'[v_2,b] \typInf \begin{array}[t]{@{}l@{}}
                    v_1:\ol{A_1}, v_2:\ol{A_2}, b:\ol{B}, x:A_1 \tensor^\pri A_2 \tensor^\pi B,
                    \\
                    a:\ol{A_2 \tensor^\pi B}, x':A_2 \tensor^\pi B
                \end{array}
            }
            \bussUn[\ruleLabel{typ-res}]{
                \vdash \pRes{ax'} ( \pOut x[v_1,a] \| \pOut x'[v_2,b] ) \typInf v_1:\ol{A_1}, v_2:\ol{A_2}, b:\ol{B}, x:A_1 \tensor^\pri A_2 \tensor^\pi B
            }
        \end{bussproof}
    \]
    As discussed before, this typing leaves the (free) names $v_1$, $v_2$, and $b$ to be accounted for by the context.

    Now let us derive the typing of the consecutive receives on $y$, i.e., of the subprocess $\pIn y(w_1,y') ; \pIn y'(w_2,y'') ; Q$.
    Because $x$ and  $y$ are dual names in $P$, the type of $y$ should be dual to the type of $x$ above:
    \[
        \begin{bussproof}
            \bussAssume{
                \vdash Q \typInf \Gamma, w_1:\ol{A_1}, w_2:\ol{A_2}, y'':\ol{B}
            }
            \bussAssume{
                \pi < \pr(\Gamma, w_1:\ol{A_1})
            }
            \bussBin[\ruleLabel{typ-recv}]{
                \vdash \pIn y'(w_2,y'') ; Q \typInf \Gamma, w_1:\ol{A_1}, y':\ol{A_2} \parr^\pi \ol{B}
            }
            \bussAssume{
                \pri < \pr(\Gamma)
            }
            \bussBin[\ruleLabel{typ-recv}]{
                \vdash \pIn y(w_1,y') ; \pIn y'(w_2,y'') ; Q \typInf \Gamma, y:\ol{A_1} \parr^\pri \ol{A_2} \parr^\pi \ol{B}
            }
        \end{bussproof}
    \]
    
    These two derivations tell us that $\pri < \pi < \pr(A_1),\pr(A_2),\pr(B),\pr(\Gamma)$.
    This way, the type system ensures that none of the sessions in $Q$ can be connected to sessions that block the sends on $x,x'$, which may leave the door open for a deadlock otherwise.
    \ifLASTn
    In \Cref{s:LASTn:LASTn:df}, \Cref{x:LASTn:LASTn:deadlocked} illustrates such a situation in the context of \LASTn and how the type system catches it.
    \fi
\end{example}

\ifLASTn\else
Note that a single application of Rule~\ruleLabel{typ-par} followed by Rule~\ruleLabel{typ-res} coincides with the usual Rule~\ruleLabel{cut} in type systems based on linear logic~\cite{conf/concur/CairesP10,conf/icfp/Wadler12}.
Without priority annotations and conditions, Rules~\ruleLabel{typ-par} and~\ruleLabel{typ-res} give rise to deadlocks, though their restriction to Rule~\ruleLabel{cut} would prevent this:

\begin{example}
    \label{x:APCP:noPrios}
    Consider the following process, arguably the paradigmatic example of a deadlock:
    \[
        Q \deq \pRes{xy} \pRes{zw} ( \pIn x(u) ; \pOut* z[u'] \cdot \0 \| \pIn w(v) ; \pOut* y[v'] \cdot \0 )
    \]
    Without priorities (and priority checks), this process can be typed using Rules~\ruleLabel{typ-par} and~\ruleLabel{typ-res}:
    \[
        \begin{bussproof}
            \def\ScoreOverhang{2pt}
            \def\defaultHypSeparation{\hskip3mm}
            \bussAx[\ruleLabel{typ-inact}]{
                \vdash \0 \typInf \emptyset
            }
            \doubleLine
            \bussUn[$\ruleLabel{typ-end}^4$]{
                \vdash \0 \typInf z:\bullet , x:\bullet , u:\bullet , u':\bullet
            }
            \bussUn[\ruleLabel{typ-send$\star$}]{
                \vdash \pOut* z[u'] \cdot \0 \typInf z:\bullet \tensor \bullet , x:\bullet , u:\bullet
            }
            \bussUn[\ruleLabel{typ-recv}]{
                \vdash \pIn x(u) ; \pOut* z[u'] \cdot \0 \typInf z:\bullet \tensor \bullet , x:\bullet \parr \bullet
            }
            \bussAx[\ruleLabel{typ-inact}]{
                \vdash \0 \typInf \emptyset
            }
            \doubleLine
            \bussUn[$\ruleLabel{typ-end}^4$]{
                \vdash \0 \typInf w:\bullet , y:\bullet , v:\bullet , v':\bullet
            }
            \bussUn[\ruleLabel{typ-send$\star$}]{
                \vdash \pOut* y[v'] \cdot \0 \typInf w:\bullet , y:\bullet \tensor \bullet , v:\bullet
            }
            \bussUn[\ruleLabel{typ-recv}]{
                \vdash \pIn w(v) ; \pOut* y[v'] \cdot \0 \typInf w:\bullet \tensor \bullet , y:\bullet \parr \bullet
            }
            \doubleLine
            \bussBin[$\begin{array}{@{}l@{}}\\\ruleLabel{typ-par}+\\\ruleLabel{typ-res}^2\end{array}$]{
                \vdash Q \typInf \emptyset
            }
        \end{bussproof}
    \]
    On the other hand, were we to restrict parallel composition and restriction to Rule~\ruleLabel{cut}, $Q$ would not be typable: Rule~\ruleLabel{cut} can only type one of the restrictions, not both.
    With priorities, $Q$ would not be typable either, due to the requirements induced by Rule~\ruleLabel{typ-recv}: (i)~the priority $\pri$ of the input on $x$ is smaller than the priority $\pi$ of the send on $z$ (left-hand side above), and (ii)~the priority $\pi$ of the input on $w$ is smaller than the priority $\pri$ of the send on $y$ (right-hand side above).
    Clearly, these requirements combined are unsatisfiable.
\end{example}
\fi\fi

\ifLASTn\else
Rules~\ruleLabel{typ-rec} and~\ruleLabel{typ-var} type recursive processes, where the former types recursive definitions and the latter types recursive calls.
\ifmpstAPCP\else
A naive Rule~\ruleLabel{typ-rec} would require a name typed $\mu X . A$ at the recursive definition to be typed simply $A$ in the recursive body.
The associated naive Rule~\ruleLabel{typ-var} would then require all names to be typed $X$, using the recursive context to make sure that the number of names concurs between recursive definition and call.
However, as the following example shows, such a combination of rules discards priority annotations to a point where it is possible to type processes that deadlock:

\begin{example}
    \label{x:APCP:naiveRec}
    Consider the processes
    \begin{align*}
        P &\deq \pRec X(x,y) ; \pOut* x[a] \cdot \pIn x(b) ; \pOut* y[c] \cdot \pIn y(d) ; \pCall X<y,x>,
        \\
        Q &\deq \pRec X(u,v) ; \pIn u(a) ; \pOut* u[b] \cdot \pIn v(c) ; \pOut* v[d] \cdot \pCall X<u,v>.
    \end{align*}
    Notice how the recursive call in $P$ swaps $x$ and $y$.
    Let us see what happens if we unfold the recursion in $P$ and $Q$:
    \begin{align*}
        P &\equiv \pOut* x[a] \cdot \pIn x(b) ; \pOut* y[c] \cdot \pIn y(d) ; \pRec X(y,x) ; \pOut* y[a] \cdot \pIn y(b) ; \pOut* x[c] \cdot \pIn x(d) ; \pCall X<x,y>
        \\
        Q &\equiv \pIn u(a) ; \pOut* u[b] \cdot \pIn v(c) ; \pOut* v[d] \cdot \pRec X(u,v) ; \pIn u(a) ; \pOut* u[b] \cdot \pIn v(c) ; \pOut* v[d] \cdot \pCall X<u,v>
    \end{align*}
    If we connect these processes on the names $x$ and $u$ and on $y$ and $v$, we can see that the second recursive definition of this process contains a deadlock: the second receive on $y$ is blocking the second send on $x$, while the second receive on $u$ (waiting for the second send on~$x$) is blocking the second send on $v$ (for which the second receive on $y$ is waiting).

    Yet, $P$ is typable using the naive typing rules described above:
    \[
        \begin{bussproof}
            \bussAx[\ruleLabel{typ-var-naive}]{
                X:2 \vdash \pCall X<y,x> \typInf x:X, y:X
            }
            \doubleLine
            \bussUn[$\cdots$]{
                X:2 \vdash \pOut* x[a] \cdot \pIn x(b) ; \pOut* y[c] \cdot \pIn y(d) ; \pCall X<y,x> \typInf x:\bullet \tensor^0 \bullet \parr^1 X, y:\bullet \tensor^2 \bullet \parr^3 X
            }
            \bussUn[\ruleLabel{typ-rec-naive}]{
                \emptyset \begin{array}[t]{@{}l@{}}
                    {} \vdash \pRec X(x,y) ; \pOut* x[a] \cdot \pIn x(b) ; \pOut* y[c] \cdot \pIn y(d) ; \pCall X<y,x> 
                    \\
                    {} \typInf x:\mu X . \bullet \tensor^0 \bullet \parr^1 X, y:\mu X . \bullet \tensor^2 \bullet \parr^3 X
                \end{array}
            }
        \end{bussproof}
    \]
    Thus, these naive rules prevent deadlock-freedom by typing.
\end{example}

\fi
\ifmpstAPCP
Rule~\ruleLabel{typ-rec} unfolds all types.
\else
The solution is for Rule~\ruleLabel{typ-rec} to unfold all types.
\fi
While unfolding, the priorities in these types are lifted by a common value, denoted $t$ in the rule, that must be greater than the top priority occurring in the original types (cf.\ \ifmpstAPCP\Cref{d:mpstAPCP:APCP:top}\else\Cref{d:APCP:top}\fi).
This makes sure that any priority requirements that come up in the typing of the recursive body of the process remain valid.
The recursive context is used to record the bodies of the original folded types.
Rule~\ruleLabel{typ-rec} then requires that the types of names are recursive on the recursion variable used for the call.
It checks that the bodies of the types concur with the types recorded in the recursive context, up to a lift by a common value $t$ (i.e., the lifter used in the application of Rule~\ruleLabel{typ-rec}).

\ifmpstAPCP\else
\begin{example}
    To see how the unfolding of types and the common lifter in Rule~\ruleLabel{typ-rec} prevents $P$ from \Cref{x:APCP:naiveRec} from being typable, let us find a typing derivation for $P$:
    \[
        \begin{bussproof}
            \def\ScoreOverhang{1pt}
            \bussAx[\ruleLabel{typ-var}]{
                X:(\bullet \tensor^0 \bullet \parr^1 X, \bullet \tensor^2 \bullet \parr^3 X) \begin{array}[t]{@{}l@{}}
                    {} \vdash \pCall X<y,x> 
                    \\
                    {} \typInf x:\mu X . \lift{t}(\bullet \tensor^0 \bullet \parr^1 X), y:\mu X . \lift{t}(\bullet \tensor^2 \bullet \parr^3 X)
                \end{array}
            }
            \doubleLine
            \bussUn[$\cdots$]{
                X:(\bullet \tensor^0 \bullet \parr^1 X, \bullet \tensor^2 \bullet \parr^3 X) \begin{array}[t]{@{}r@{}l@{}}
                    {} \vdash {}
                    & \pOut* x[a] \cdot \pIn x(b) ; \pOut* y[c] \cdot \pIn y(d) ; \pCall X<y,x>
                    \\
                    {} \typInf {}
                    & x:\bullet \tensor^0 \bullet \parr^1 (\mu X . \lift{t}( \bullet \tensor^0 \bullet \parr^1 X)),
                    \\
                    & y:\bullet \tensor^2 \bullet \parr^3 (\mu X . \lift{t}( \bullet \tensor^2 \bullet \parr^3 X))
                \end{array}
            }
            \bussUn[\ruleLabel{typ-rec}]{
                \emptyset \vdash \pRec X(x,y) ; \pOut* x[a] \cdot \pIn x(b) ; \pOut* y[c] \cdot \pIn y(d) ; \pCall X<y,x> \typInf x:\mu X . \bullet \tensor^0 \bullet \parr^1 X, y:\mu X . \bullet \tensor^2 \bullet \parr^3 X
            }
        \end{bussproof}
    \]
    The application of Rule~\ruleLabel{typ-var} at the top here is invalid: it is impossible to find a lifter~$t$ that matches the priorities in the type of $x$ with those in the second type assigned to $X$, while simultaneously doing the same for the type of $y$ and the first type assigned to $X$.
    However, if the call were $\pCall X<x,y>$, the application of Rule~\ruleLabel{typ-var} would be valid.
\end{example}

\begin{remark}[Comparison to \protect\PCP]
    \label{r:APCP:simplerPrios}
    Consider the typing rules for sending and selection in \PCP, in which both are blocking prefixes and do not involve continuation-passing.
    Note that \PCP does not have recursion, and we thus omit the recursive context~$\Omega$.
    \begin{mathpar}
        \begin{bussproof}[typ-send-PCP]
            \bussAssume{
                \vdash P \typInf \Gamma, y:A, x:B
            }
            \bussAssume{
                \pri < \pr(\Gamma)
            }
            \bussBin{
                \vdash \pOut* x[y] ; P \typInf \Gamma, x:A \tensor^\pri B
            }
        \end{bussproof}
        \and
        \begin{bussproof}[typ-sel-PCP]
            \bussAssume{
                \vdash P \typInf \Gamma, x:A_j
            }
            \bussAssume{
                j \in I
            }
            \bussAssume{
                \pri < \pr(\Gamma)
            }
            \bussTern{
                \vdash \pSel* x < j ; P \typInf \Gamma, x:\oplus^\pri \{ i : A_i \}_{i \in I}
            }
        \end{bussproof}
    \end{mathpar}
    These rules are similar to the Rules~\ruleLabel{typ-send$\ast$} and~\ruleLabel{typ-sel$\ast$} in \Cref{f:APCP:typingRules} for the sugared syntax introduced in \Cref{n:APCP:sugar}; the differences are twofold.
    First, the semicolons `\,$;$\,' indicate that sends and selections are indeed blocking.
    Second, the rules compare the priority of the send/selection to the priorities in the context, whereas our rules compare this priority to the priorities of the continuation of the send/selection.
\end{remark}
\fi\fi

\ifLASTn
The
\else\ifmpstAPCP
The
\else
As anticipated, the
\fi\fi
binding of sends and selections to continuation processes (\ifLASTn\Cref{n:LASTn:APCP:sugar}\else\ifmpstAPCP\Cref{n:mpstAPCP:APCP:sugar}\else\Cref{n:APCP:sugar}\fi\fi) is derivable in \APCP.
\ifLASTn\else\ifmpstAPCP\else
The corresponding typing rules in \Cref{f:APCP:typingRules} (bottom) are admissible using Rules~\ruleLabel{typ-par} and~\ruleLabel{typ-res}.
It is not necessary to include rules for the sugared receive and branching in \Cref{n:APCP:sugar}, because they rely on $\alpha$-renaming only.
\fi\fi
\ifLASTn\Cref{f:LASTn:APCP:typing}\else\ifmpstAPCP\Cref{f:mpstAPCP:APCP:typing}\else\Cref{f:APCP:typingRules}\fi\fi~(bottom) also includes an admissible Rule~\ruleLabel{typ-lift} that lifts a process' priorities.

\begin{theorem}
    \ifLASTn \label{t:LASTn:APCP:derivable}
    \else\ifmpstAPCP \label{t:mpstAPCP:APCP:derivable}
    \else \label{t:APCP:derivable} \fi\fi
    Rules~\ruleLabel{typ-send$\ast$} and~\ruleLabel{typ-sel$\ast$} in \ifLASTn\Cref{f:LASTn:APCP:typing}\else\ifmpstAPCP\Cref{f:mpstAPCP:APCP:typing}\else\Cref{f:APCP:typingRules}\fi\fi~(bottom) are derivable\ifLASTn\else, and Rule~\ruleLabel{typ-lift} in \ifLASTn\Cref{f:LASTn:APCP:typing}\else\ifmpstAPCP\Cref{f:mpstAPCP:APCP:typing}\else\Cref{f:APCP:typingRules}\fi\fi~(bottom) is admissible\fi.
\end{theorem}

\ifLASTn\else\ifmpstAPCP\else
\begin{figure}[t]
    \begin{align*}
        & \begin{bussproof}[typ-send$\ast$]
            \bussAssume{
                \vdash P \typInf \Gamma, y:A, x:B
            }
            \bussAssume{
                \pri < \pr(A),\pr(B)
            }
            \bussBin{
                \vdash \pOut* x[y] \cdot P \typInf \Gamma, x:A \tensor^\pri B
            }
        \end{bussproof}
        \\
        & \Rightarrow
        \\
        & \begin{bussproof}
            \bussAssume{
                \pri < \pr(A),\pr(B)
            }
            \bussUn[\ruleLabel{typ-send}]{
                \vdash \pOut x[a, b] \typInf x:A \tensor^\pri B, a:\ol{A}, b:\ol{B}
            }
            \bussAssume{
                \vdash P \{ z/x \} \typInf \Gamma, y:A, z:B
            }
            \bussBin[\ruleLabel{typ-par}]{
                \vdash \pOut x[a, b] \| P \{ z/x \} \typInf \Gamma, x:A \tensor^\pri B, y:A, a:\ol{A}, z:B, b:\ol{B}
            }
            \doubleLine
            \bussUn[$2\times$\ruleLabel{typ-res}]{
                \vdash \underbrace{
                    \pRes{ya} \pRes{zb} ( \pOut x[a, b] \| P \{ z/x \} )
                }_{\pOut* x[y] \cdot P ~\text{(cf.\ \Cref{n:APCP:sugar})}} \typInf \Gamma, x:A \tensor^\pri B
            }
        \end{bussproof}
        \\[5pt]
        & \begin{bussproof}[typ-sel$\ast$]
            \bussAssume{
                \vdash P \typInf \Gamma, x:A_j
            }
            \bussAssume{
                j \in I
            }
            \bussAssume{
                \pri < \pr(A_j)
            }
            \bussTern{
                \vdash \pSel* x < j \cdot P \typInf \Gamma, x:\oplus^\pri \{ i : A_i \}_{i \in I}
            }
        \end{bussproof}
        \\
        & \Rightarrow
        \\
        & \begin{bussproof}
            \bussAssume{
                j \in I
            }
            \bussUn[\ruleLabel{typ-sel}]{
                \vdash \pSel x[b] < j \typInf x:\oplus^\pri \{ i : A_i \}_{i \in I}, b:\ol{A_j}
            }
            \bussAssume{
                \vdash P \{ z/x \} \typInf \Gamma, z:A_j
            }
            \bussBin[\ruleLabel{typ-par}]{
                \vdash \pSel x[b] < j \| P \{ z/x \} \typInf \Gamma, x:\oplus^\pri \{ i : A_i \}_{i \in I}, z:A_j, b:\ol{A_j}
            }
            \bussUn[\ruleLabel{typ-res}]{
                \vdash \underbrace{
                    \pRes{zb} ( \pSel x[b] < j \| P \{ z/x \} )
                }_{\pSel* x < j \cdot P ~\text{(cf.\ \Cref{n:APCP:sugar})}} \typInf \Gamma, x:\oplus^\pri \{ i : A_i \}_{i \in I}
            }
        \end{bussproof}
    \end{align*}

    \caption{Proof that Rules~\protect\ruleLabel{typ-send$\ast$} and~\protect\ruleLabel{typ-sel$\ast$} are derivable (cf.\ \Cref{t:APCP:derivable}).}
    \label{f:APCP:derivable}
\end{figure}

\begin{sketch}
    We show the derivability of Rules~\ruleLabel{typ-send$\ast$} and~\ruleLabel{typ-sel$\ast$} by giving their derivations in \Cref{f:APCP:derivable} (omitting the recursive context~$\Omega$).
    Rule~\ruleLabel{typ-lift} is admissible, because $\Omega \vdash P \typInf \Gamma$ implies $\Omega \vdash P \typInf \lift{t} \Gamma$ (cf.\ Dardha and Gay~\cite{conf/fossacs/DardhaG18}), by simply increasing all priorities in the derivation of $P$ by~$t$.
\end{sketch}

\noindent
\Cref{t:APCP:derivable} highlights how \APCP's asynchrony uncovers a more primitive, lower-level view of message-passing.
In the next subsection we discuss deadlock-freedom, which follows from a correspondence between reduction and the transformation of typing derivations to eliminate applications of Rule~\ruleLabel{typ-res}.
In the case of \APCP, this requires care: binding sends and selections to continuation processes leads to applications of Rule~\ruleLabel{typ-res} that do not immediately correspond to reductions.

\begin{remark}[Comparison To DeYoung \etal]
    Our rules for sending and selection are axiomatic, whereas DeYoung \etal's are in the form of Rules~\ruleLabel{typ-send$\ast$} and~\ruleLabel{typ-sel$\ast$}, even though their sends and selections are parallel atomic prefixes as well~\cite{conf/csl/DeYoungCPT12}.
    While our type system is based on classical linear logic (with single-sided judgments), their type system is based on \emph{intuitionistic} linear logic.
    As a result, their typing judgments are two-sided, restriction involves a single name, and there are two rules (right and left) per connective.
    This way, for instance, their right rules for sending and selection are as follows (again omitting the recursive context $\Omega$ for a lack of recursion):
    \begin{mathpar}
        \begin{bussproof}[typ-send-R]
            \bussAssume{
                \Gamma \vdash P \typInf y:A
            }
            \bussAssume{
                \Delta \vdash Q \typInf x':B
            }
            \bussBin{
                \Gamma, \Delta \vdash \pRes{y} \pRes{x'} ( \pOut x[y,x'] \| P \| Q ) \typInf x:A \tensor B
            }
        \end{bussproof}
        \and
        \begin{bussproof}[typ-sel-R]
            \bussAssume{
                \Gamma \vdash P \typInf x':A_j
            }
            \bussAssume{
                j \in I
            }
            \bussBin{
                \Gamma \vdash \pRes{x'} ( \pSel x[x'] < j \| P ) \typInf x:\oplus \{ i : A_i \}_{i \in I}
            }
        \end{bussproof}
        ~$\qedsymbol$
    \end{mathpar}
    \def\qed{}
\end{remark}

\newpage
The following result is important: it shows that the type system of \APCP is complete with respect to types, i.e., every syntactical type has a well-typed process.

\begin{proposition}
    \ifmpstAPCP \label{p:mpstAPCP:charProc}
    \else \label{p:APCP:charProc} \fi
    Given a type $A$, there exists $\Omega \vdash P \typInf x:A$.
\end{proposition}

\begin{sketch}
    By constructing $P$ from the structure of $A$.
    To this end, we define characteristc processes: a function $\chr^x(A)$ that constructs a process that performs the behavior described by $A$ on the name $x$.
    \begin{align*}
        \chr^x(A \tensor^\pri B)
        &\deq
        \pOut* x[y] \cdot ( \chr^y(A) \| \chr^x(B) )
        &
        \chr^x(\bullet)
        &\deq
        \0
        \\
        \chr^x(A \parr^\pri B)
        &\deq
        \pIn x(y) ; ( \chr^y(A) \| \chr^x(B) )
        &
        \chr^x(\mu X . A)
        &\deq
        \pRec X(x) ; \chr^x(A)
        \\
        \chr^x(\oplus^\pri \{ i : A_i \}_{i \in I})
        &\deq
        \pSel* x < j \cdot \chr^x(A_j)
        ~ \text{(any $j \in I$)}
        &
        \chr^x(X)
        &\deq
        \pCall X<x>
        \\
        \chr^x(\&^\pri \{ i : A_i \}_{i \in I})
        &\deq
        \pBra* x > {\{ i : \chr^x(A_i) \}_{i \in I}}
    \end{align*}
    For finite types, it is obvious that $\emptyset \vdash \chr^x(A) \typInf x:A$.
    For simplicity, we omit details about recursive types, which require unfolding.
    For closed, recursive types, the thesis is obvious as well: $\emptyset \vdash \chr^x(\mu X . A) \typInf x: \mu X . A$.
\end{sketch}
\fi\fi

\subsection{Type Preservation and Deadlock-freedom}
\label{s:LASTn:APCP:APCP:results}

\sloppy
Well-typed processes satisfy protocol fidelity, communication safety, and deadlock-freedom.
The former two of these properties follow from \emph{type preservation}, which ensures that \APCP's semantics preserves typing.
\ifLASTn\else\ifmpstAPCP\else
In contrast to Caires and Pfenning~\cite{conf/concur/CairesP10} and Wadler~\cite{conf/icfp/Wadler12}, where type preservation corresponds to the elimination of (top-level) applications of Rule~\ruleLabel{cut}, in \APCP it corresponds to the elimination of (top-level) applications of Rule~\ruleLabel{type-res}.
\ifthesis
Proofs are detailed in\ifappendix \Cref{ac:APCP}\else~\cite{thesis/vdHeuvel24ex}\fi.
\fi

\APCP's semantics consists of reduction and structural congruence.
Since the former relies on the latter, we first need to show type preservation for structural congruence, i.e., \emph{subject congruence}.
The structural congruence rule that unfolds recursive definitions requires care, because the types of the unfolded process are also unfolded:

\begin{figure}[t]
    \[
        \begin{bussproof}
            \bussAx[\ruleLabel{typ-var}]{
                X:(\bullet \parr^2 X, \bullet \parr^3 X) \vdash \pCall X<x,y> \typInf x:\mu X . \underbrace{\bullet \parr^6 X}_{\lift{4} (\bullet \parr^2 X)}, y:\mu X . \underbrace{\bullet \parr^7 X}_{\lift{4} (\bullet \parr^3 X)}
            }
            \bussUn[\ruleLabel{typ-end}]{
                X:(\bullet \parr^2 X, \bullet \parr^3 X) \vdash \pCall X<x,y> \typInf x:\mu X . \bullet \parr^6 X, y:\mu X . \bullet \parr^7 X, b':\bullet
            }
            \bussUn[\ruleLabel{typ-recv}]{
                X:(\bullet \parr^2 X, \bullet \parr^3 X) \vdash \pIn y(b') ; \pCall X<x,y> \typInf x:\mu X . \bullet \parr^6 X, y:\bullet \parr^3 \mu X . \bullet \parr^7 X
            }
            \bussUn[\ruleLabel{typ-end}]{
                X:(\bullet \parr^2 X, \bullet \parr^3 X) \vdash \pIn y(b') ; \pCall X<x,y> \typInf x:\mu X . \bullet \parr^6 X, y:\bullet \parr^3 \mu X . \bullet \parr^7 X, a':\bullet
            }
            \bussUn[\ruleLabel{typ-recv}]{
                X:(\bullet \parr^2 X, \bullet \parr^3 X)
                \begin{array}[t]{@{}r@{}l@{}}
                    {} \vdash {}
                    & \pIn x(a') ; \pIn y(b') ; \pCall X<x,y>
                    \\
                    {} \typInf {}
                    & x:\underbrace{
                        \bullet \parr^2 \mu X . \bullet \parr^6 X
                    }_{\unfold^4(\mu X . \bullet \parr^2 X)}, y:\underbrace{
                        \bullet \parr^3 \mu X . \bullet \parr^7 X
                    }_{\unfold^4(\mu X . \bullet \parr^3 X)}
                \end{array}
            }
            \bussUn[\ruleLabel{typ-rec}]{
                \emptyset \vdash \pRec X(x,y) ; \pIn x(a') ; \pIn y(b') ; \pCall X<x,y> \typInf x:\mu X . \bullet \parr^2 X, y:\mu X . \bullet \parr^3 X
            }
            \bussUn[\ruleLabel{typ-end}]{
                \emptyset \vdash \pRec X(x,y) ; \pIn x(a') ; \pIn y(b') ; \pCall X<x,y> \typInf x:\mu X . \bullet \parr^2 X, y:\mu X . \bullet \parr^3 X, b:\bullet
            }
            \bussUn[\ruleLabel{typ-recv}]{
                \emptyset \vdash \pIn y(b) ; \pRec X(x,y) ; \pIn x(a') ; \pIn y(b') ; \pCall X<x,y> \typInf x:\mu X . \bullet \parr^2 X, y:\bullet \parr^1 \mu X . \bullet \parr^3 X
            }
            \bussUn[\ruleLabel{typ-end}]{
                \emptyset \vdash \pIn y(b) ; \pRec X(x,y) ; \pIn x(a') ; \pIn y(b') ; \pCall X<x,y> \typInf x:\mu X . \bullet \parr^2 X, y:\bullet \parr^1 \mu X . \bullet \parr^3 X, a:\bullet
            }
            \bussUn[\ruleLabel{typ-recv}]{
                \emptyset
                \begin{array}[t]{@{}r@{}l@{}}
                    {} \vdash {}
                    & \pIn x(a) ; \pIn y(b) ; \pRec X(x,y) ; \pIn x(a') ; \pIn y(b') ; \pCall X<x,y>
                    \\
                    {} \typInf {}
                    & x:\underbrace{
                        \bullet \parr^0 \mu X . \bullet \parr^2 X
                    }_{\unfold^2(\mu X . \bullet \parr^0 X)}, y:\underbrace{
                        \bullet \parr^1 \mu X . \bullet \parr^3 X
                    }_{\unfold^2(\mu X . \bullet \parr^1 X)}
                \end{array}
            }
        \end{bussproof}
    \]
    \caption{Attempt at a typing derivation (cf.\ \Cref{x:APCP:unfrel}).}\label{f:APCP:x:unfrel}
\end{figure}

\begin{example}
    \label{x:APCP:unfrel}
    Consider the following typed recursive definition:
    \[
        \emptyset \vdash \pRec X(x,y) ; \pIn x(a) ; \pIn y(b) ; \pCall X<x,y> \typInf x:\mu X . \bullet \parr^0 X, y:\mu X . \bullet \parr^1 X
    \]
    \Cref{f:APCP:x:unfrel} illustrates an attempt at typing the unfolding of this process.
    Clearly, the typing of the unfolded process is not the same as the initial type, but they are equal up to unfolding of types.
    Note that the application of Rule~\ruleLabel{typ-rec} unfolds the types with a common lifter of 4, because it needs to be larger than the top priority in the types before unfolding which is 3.
\end{example}
\fi

\ifmpstAPCP Type \else Hence, type \fi preservation holds \emph{up to unfolding}.
To formalize this, we define the relation $(\vdash P \typInf \Gamma) \unfrel \Gamma'$, which says that $\Gamma$ and $\Gamma'$ are equal up to (un)folding of recursive types consistent with the typing of $P$ under $\Gamma$:

\newpage
\begin{definition}
    \ifmpstAPCP \label{d:mpstAPCP:APCP:unfrel}
    \else \label{d:APCP:unfrel} \fi
    We define an asymmetric relation between a typed process $(\Omega \vdash P \typInf \Gamma)$ and a typing context $\Gamma'$, denoted $(\Omega \vdash P \typInf \Gamma) \unfrel \Gamma'$.
    The relation is defined by the inference rules in \ifmpstAPCP\Cref{f:mpstAPCP:unfrel}\else\Cref{f:APCP:unfrel}\fi, where each rule implicitly requires that $\Omega \vdash P \typInf \Gamma$ is a valid typing derivation by the rules in \ifmpstAPCP\Cref{f:mpstAPCP:APCP:typing}\else\Cref{f:APCP:typingRules}\fi.

    We write $(\Omega \vdash P \typInf \Gamma) \unfrel* (\Omega' \vdash Q \typInf \Gamma')$ if $(\Omega \vdash P \typInf \Gamma) \unfrel \Gamma'$ and $(\Omega' \vdash Q \typInf \Gamma') \unfrel \Gamma$.
\end{definition}

\begin{figure}[p]
    \begin{mathpar}
        \begin{bussproof}[unf-fold]
            \bussAssume{
                (\Omega \vdash P \big\{ \big( \pRec X(\tilde{y}) ; P \{ \tilde{y}/\tilde{z} \} \big) / \pCall X<\tilde{y}> \big\} \typInf \tilde{z}:\tilde{U}) \unfrel \tilde{z}:\tilde{U'}
            }
            \bussAssume{
                \forall U'_i \in \tilde{U'};~ U'_i = \unfold^t(\mu X . A'_i)
            }
            \bussBin{
                (\Omega \vdash P \big\{ \big( \pRec X(\tilde{y}) ; P \{ \tilde{y}/\tilde{z} \} \big) / \pCall X<\tilde{y}> \big\} \typInf \tilde{z}:\tilde{U}) \unfrel \tilde{z}:\widetilde{\mu X . A'}
            }
        \end{bussproof}
        \and
        \begin{bussproof}[unf-unf]
            \bussAssume{
                (\Omega \vdash \pRec X(\tilde{z}) ; P \typInf \tilde{z}:\widetilde{\mu X . A}) \unfrel \tilde{z}:\widetilde{\mu X . A'}
            }
            \bussAssume{
                \forall U'_i \in \tilde{U'}.~ U'_i = \unfold^t(\mu X . A'_i)
            }
            \bussBin{
                (\Omega \vdash \pRec X(\tilde{z}) ; P \typInf \tilde{z}:\widetilde{\mu X . A}) \unfrel \tilde{z}:\tilde{U'}
            }
        \end{bussproof}
        \\ \dashes \\
        \begin{bussproof}[unf-send]
            \bussAx{
                (\Omega \vdash \pOut x[y,z] \typInf x:A \tensor^\pri B, y:\ol{A}, z:\ol{B}) \unfrel x:A \tensor^\pri B, y:\ol{A}, z:\ol{B}
            }
        \end{bussproof}
        \and
        \begin{bussproof}[unf-recv]
            \bussAssume{
                (\Omega \vdash P \typInf \Gamma, y:A, z:B) \unfrel \Gamma', y:A', z:B'
            }
            \bussUn{
                (\Omega \vdash \pIn x(y, z) ; P \typInf \Gamma, x:A \parr^\pri B) \unfrel \Gamma', x:A' \parr^\pri B'
            }
        \end{bussproof}
        \and
        \begin{bussproof}[unf-sel]
            \bussAx{
                (\Omega \vdash \pSel x[z] < j \typInf x:\oplus^\pri \{ i : A_i \}_{i \in I}, z:\ol{A_j}) \unfrel x:\oplus^\pri \{ i : A_i \}_{i \in I}, z:\ol{A_j}
            }
        \end{bussproof}
        \and
        \begin{bussproof}[unf-bra]
            \bussAssume{
                \forall i \in I.~ (\Omega \vdash P_i \typInf \Gamma, z:A_i) \unfrel \Gamma', z:A'_i
            }
            \bussUn{
                (\Omega \vdash \pBra x(z) > {\{ i : P_i \}_{i \in I}} \typInf \Gamma, x:\&^\pri \{ i : A_i \}_{i \in I}) \unfrel \Gamma', x:\&^\pri \{ i : A'_i \}_{i \in I}
            }
        \end{bussproof}
        \and
        \begin{bussproof}[unf-end]
            \bussAssume{
                (\Omega \vdash P \typInf \Gamma) \unfrel \Gamma'
            }
            \bussUn{
                (\Omega \vdash P \typInf \Gamma, x:\bullet) \unfrel \Gamma', x:\bullet
            }
        \end{bussproof}
        \and
        \begin{bussproof}[unf-par]
            \bussAssume{
                (\Omega \vdash P \typInf \Gamma) \unfrel \Gamma'
            }
            \bussAssume{
                (\Omega \vdash Q \typInf \Delta) \unfrel \Delta'
            }
            \bussBin{
                (\Omega \vdash P \| Q \typInf \Gamma, \Delta) \unfrel \Gamma' , \Delta'
            }
        \end{bussproof}
        \and
        \begin{bussproof}[unf-res]
            \bussAssume{
                (\Omega \vdash P \typInf \Gamma, x:A, y:\ol{A}) \unfrel \Gamma', x:A', y:\ol{A'}
            }
            \bussUn{
                (\Omega \vdash \pRes{xy} P \typInf \Gamma) \unfrel \Gamma'
            }
        \end{bussproof}
        \and
        \begin{bussproof}[unf-inact]
            \bussAx{
                (\Omega \vdash \0 \typInf \emptyset) \unfrel \emptyset
            }
        \end{bussproof}
        \and
        \begin{bussproof}[unf-fwd]
            \bussAx{
                (\Omega \vdash \pFwd [x<>y] \typInf x:\ol{A}, y:A) \unfrel x:\ol{A}, y:A
            }
        \end{bussproof}
        \and
        \begin{bussproof}[unf-rec]
            \bussAssume{
                (\Omega, X:{\tilde{A}} \vdash P \typInf \tilde{z}:\tilde{U}) \unfrel \tilde{z}:\tilde{U'}
            }
            \bussUn{
                (\Omega \vdash \pRec X(\tilde{z}) ; P \typInf \tilde{z}:\widetilde{\mu X . A}) \unfrel \tilde{z}:\widetilde{\mu X . A'}
            }
        \end{bussproof}
        \and
        \begin{bussproof}[unf-var]
            \bussAx{
                (\Omega, X:\tilde{A} \vdash \pCall X<\tilde{z}> \typInf \tilde{z}:\tilde{U}) \unfrel \tilde{z}:\tilde{U}
            }
        \end{bussproof}
    \end{mathpar}

    \caption{Inference rules for \ifmpstAPCP\Cref{d:mpstAPCP:APCP:unfrel}\else\Cref{d:APCP:unfrel}\fi.}
    \ifmpstAPCP \label{f:mpstAPCP:unfrel}
    \else \label{f:APCP:unfrel} \fi
\end{figure}

The most important rules of \ifmpstAPCP \Cref{f:mpstAPCP:unfrel} \else \Cref{f:APCP:unfrel} \fi are Rules~\ruleLabel{unf-unf} and~\ruleLabel{unf-fold} (above the line), as they relate unfolding and folding; the other rule (below the line) follow the typing rules in \ifmpstAPCP\Cref{f:mpstAPCP:APCP:typing}\else\Cref{f:APCP:typingRules}\fi.
Note that the rules in \ifmpstAPCP\Cref{f:mpstAPCP:unfrel}\else\Cref{f:APCP:unfrel}\fi require no priority requirements, as they are covered by the implicit validity of the derivation of $\Omega \vdash P \typInf \Gamma$.

\ifmpstAPCP\else
\begin{proposition}[Transitivity of $\protect\unfrel*$]
    \label{p:APCP:unfrelTrans}
    If (i)~$(\Omega_1 \vdash P_1 \typInf \Gamma_1) \unfrel* (\Omega_2 \vdash P_2 \typInf \Gamma_2)$ and (ii)~\mbox{$(\Omega_2 \vdash P_2 \typInf \Gamma_2) \unfrel* (\Omega_3 \vdash P_3 \typInf \Gamma_3)$}, then $(\Omega_1 \vdash P_1 \typInf \Gamma_1) \unfrel* (\Omega_3 \vdash P_3 \typInf \Gamma_3)$.
\end{proposition}

\begin{sketch}
    Since $\unfrel*$ is reflexive by definition, it suffices to show that (i)~$(\Omega_1 \vdash P_1 \typInf \Gamma_1) \unfrel \Gamma_2$ and (ii)~$(\Omega_2 \vdash P_2 \typInf \Gamma_2) \unfrel \Gamma_3$ imply $(\Omega_1 \vdash P_1 \typInf \Gamma_1) \unfrel \Gamma_3$.
    Assumptions~(i) and~(ii) relate $\Gamma_1,\Gamma_2,\Gamma_3$ by applications of Rules~\ruleLabel{unf-fold} and~\ruleLabel{unf-unf}.
    Clearly, the relation between $\Gamma_1$ and $\Gamma_3$ can then be derived directly by combining the applications of these rules.
\end{sketch}
\fi
\fi

\ifLASTn
\begin{theorem}[Type Preservation for \APCP]
    \label{t:LASTn:APCP:tp}
    If $\vdash P \typInf \Gamma$ and $P \equiv Q$ or $P \redd Q$, then $\vdash Q \typInf \Gamma$.
\end{theorem}
\else
\ifmpstAPCP\else
Subject congruence additionally requires the following property of typing derivations involving recursion variables:

\begin{proposition}
    \label{p:APCP:recPreserved}
    Suppose given a process $P$ and a derivation $\pi$ of $\Omega, X:{(A_l)}_{l \in L} \vdash P \typInf \Gamma$.
    In every step in $\pi$, the assignment $X:{(A_l)}_{l \in L}$ in the recursive context is not modified nor removed, and no other assignments on the recursion variable $X$ are added.
\end{proposition}

\begin{sketch}
    By induction on the height of the derivation $\pi$: no typing rule eliminates assignments from the recursive context, and no typing rule changes the types assigned to a recursion variable in the recursive context.
    Moreover, in our type system, any rule adding an assignment to a context implicitly assumes that the newly introduced name or recursion variable is not in the context yet.
    Hence, an assignment on $X$ cannot be added by any rule.
\end{sketch}

Now we can state and prove subject congruence, up to the relation in \Cref{d:APCP:unfrel}:

\begin{restatable}[Subject Congruence]{theorem}{tAPCPSC}
    \label{t:APCP:sc}
    If $\Omega \vdash P \typInf \Gamma$ and $P \equiv Q$, then there exists $\Gamma'$ such that $\Omega \vdash Q \typInf \Gamma'$ and $(\Omega \vdash P \typInf \Gamma) \unfrel* (\Omega \vdash Q \typInf \Gamma')$.
\end{restatable}

\ifthesis
\begin{sketch}
    By induction on the derivation of $P \equiv Q$.
    The inductive cases induced by process contexts follow from the IH straightforwardly.
    Base cases, induced by the axioms in \Cref{d:APCP:APCP:strcong}, are straightforward with $\Gamma' = \Gamma$, except for (un)folding recursion.
    In the case of unfolding, we derive the typing of the unfolded process by substituting relevant applications of Rule~\ruleLabel{typ-var} by copies of the entire recursive definition, appropriately renamed and weakened.
    The resulting process is typed under $\Gamma'$.
    Hence, we need to show the relation between the original and the derived typing following \Cref{d:APCP:unfrel}.
    This is straightforward since $\Gamma'$ is simply an unfolding of $\Gamma$.
    The case of folding is analogous.
\end{sketch}
\else
\begin{proof}
    By induction on the derivation of $P \equiv Q$.
    The cases for the structural rules follow from the IH directly.
    All base cases are straightforward, except the cases of unfolding and folding.
    We detail unfolding, where $P = \pRec X(\tilde{z}) ; P'$ and $Q = P' \big\{ \big( \pRec X(\tilde{y}) ; P' \{ \tilde{y}/\tilde{z} \} \big) / \pCall X<\tilde{y}> \big\}$.
    The other direction, folding, is similar.

    By type inversion of Rule~\ruleLabel{typ-rec},
    \begin{align*}
        \Omega &\vdash P \typInf \Omega; \tilde{z}:\widetilde{\mu X . A}
        \\
        \Omega, X:\tilde{A} &\vdash P' \typInf \tilde{z}:\tilde{U}
    \end{align*}
    where there exists $t \in \mbb{N} > \prtop(\widetilde{\mu X . A})$ such that each
    \[
        U_i = \unfold^t(\mu X . A_i) = A_i \{ \big( \mu X . \lift{t} A_i \big) / X \}.
    \]
    Let $\pi$ denote the typing derivation of $P'$.

    $P'$ may contain $m \geq 0$ recursive calls $\pCall X<\tilde{y}>$.
    For each call, there is an application of Rule~\ruleLabel{typ-var} in $\pi$.
    We uniquely identify each such application of Rule~\ruleLabel{typ-var} as $\rho_j$ for~\mbox{$0 < j \leq m$}.
    Since \APCP only allows tail-recursion, these multiple recursive calls can only occur inside branches on names not in $\tilde{z}$, so the common lifter of each $\rho_j$ must be~$t$.
    \begin{align*}
        \rho_j \left\{ \qquad
            \begin{bussproof}
                \bussAx[\ruleLabel{typ-var}]{
                    \Omega'_j, X:\tilde{A} \vdash \pCall X<\tilde{y}_j> \typInf \tilde{y}_j:\widetilde{\mu X . \lift{t} A}
                }
            \end{bussproof}
        \right.
    \end{align*}
    where $\Omega'_j \supseteq \Omega$, because rules applied in $\pi$ can only add assignments to the recursive context.
    Recall that $\tilde{y}_j:\widetilde{\mu X . \lift{t} A}$ is notation for $y_{j_1}:\mu X . \lift{t} A_1 , \ldots , y_{j_k}:\mu X . \lift{t} A_k$, as introduced in \Cref{s:APCP:APCP:types}.

    The unfolding of a recursive process replaces each recursive call by a copy of the recursive definition.
    Hence, to find a typing derivation for $Q$ we proceed as follows, for each $0 < j \leq m$:
    \begin{enumerate}

        \item
            We obtain a derivation $\pi'_j$ of $P' \{ \tilde{y}_j/\tilde{z} \}$ from $\pi$ by substituting $\tilde{z}$ for $\tilde{y}_j$ (while avoiding capturing names) and by lifting all priorities by $t$ (including the priorities of $X$ in the recursive context of $\rho_j$).

        \item
            We apply Rule~\ruleLabel{typ-rec} on $X$ in the conclusion of $\pi'_j$, resulting in a new typing derivation $\rho'_j$:
            \begin{align*}
                \rho'_j \left\{ \qquad
                    \begin{bussproof}(c)
                        \bussAssume{
                            \pi'_j
                        }
                        \noLine
                        \bussUn{
                            \Omega'_j, X:\widetilde{\lift{t} A} \vdash P' \{ \tilde{y}_j/\tilde{z} \} \typInf \tilde{y}_j:\widetilde{\lift{t} U}
                        }
                        \bussUn[\ruleLabel{typ-rec}]{
                            \Omega'_j \vdash \pRec X(\tilde{y}_j) ; P' \{ \tilde{y}_j/\tilde{z} \} \typInf \tilde{y}_j:\widetilde{\mu X . \lift{t} A}
                        }
                    \end{bussproof}
                \right.
            \end{align*}

    \end{enumerate}

    For every $0 < j \leq m$, the context of the conclusion of $\rho'_j$ coincides with the context of the conclusion of $\rho_j$, up to the assignment $X:\widetilde{\lift{t} A}$ in the recursive context not present in~$\rho'_j$.
    We intend to obtain from $\pi$ a new derivation $\pi'$ by replacing each $\rho_j$ with $\rho'_j$.
    By \Cref{p:APCP:recPreserved}, the fact that each $\rho'_j$ is missing the assignment to $X$ in the recursive context does not influence the steps in $\pi$.
    Hence, we can indeed obtain such a derivation~$\pi'$:
    \[
        \begin{bussproof}
            \bussAssume{
                \pi'
            }
            \noLine
            \bussUn{
                \Omega \vdash \underbrace{
                    P' \big\{ \big( \pRec X(\tilde{y}) ; P' \{ \tilde{y}/\tilde{z} \} \big) / \pCall X<\tilde{y}> \big\}
                }_{Q} \typInf \tilde{z}: \tilde{U}
            }
        \end{bussproof}
    \]
    where $\forall U_i \in \tilde{U}.~ U_i = \unfold^t(A_i)$.

    Let $\Gamma' = \tilde{z}:\tilde{U}$.
    We have $\Omega \vdash Q \typInf \Gamma'$.
    To prove the thesis, we have to show that $\Gamma$ and $\Gamma'$ are equal up to unfolding.
    Following the typing rules applied in the derivation of \mbox{$\Omega \vdash Q \typInf \tilde{z}:\tilde{U}$},  we have $(\Omega \vdash Q \typInf \tilde{z}:\tilde{U}) \unfrel \tilde{z}:\tilde{U}$ (cf.\ \Cref{d:APCP:unfrel}).
    Then, by Rule~\ruleLabel{unf-fold}, \mbox{$(\Omega \vdash Q \typInf \tilde{z}:\tilde{U}) \unfrel \tilde{z}:\widetilde{\mu X . A}$}.
    Similarly, we have $(\Omega \vdash P \typInf \tilde{z}:\widetilde{\mu X . A}) \unfrel \tilde{z}:\widetilde{\mu X . A}$, so by Rule~\ruleLabel{unf-unf}, $(\Omega \vdash P \typInf \tilde{z}:\widetilde{\mu X . A}) \unfrel \tilde{z}:\tilde{U}$.
    Hence, $(\Omega \vdash P \typInf \Gamma) \unfrel* (\Omega \vdash Q \typInf \Gamma')$, proving the thesis.
\end{proof}
\fi

\begin{figure}[t]
    \begin{align*}
        & \begin{bussproof}
            \bussAx[\ruleLabel{typ-send}]{
                \vdash \pOut x[a,b] \typInf x:A \tensor^\pri B, a:\ol{A}, b:\ol{B}
            }
            \bussAssume{
                \vdash P \typInf \Gamma, z:\ol{A}, y':\ol{B}
            }
            \bussUn[\ruleLabel{typ-recv}]{
                \vdash \pIn y(z,y') ; P \typInf \Gamma, y:\ol{A} \parr^\pri \ol{B}
            }
            \doubleLine
            \bussBin[\ruleLabel{typ-par}+\ruleLabel{typ-res}]{
                \vdash \pRes{xy} ( \pOut x[a,b] \| \pIn y(z,y') ; P ) \typInf \Gamma, a:\ol{A}, b:\ol{B}
            }
        \end{bussproof}
        \\
        & \redd
        \\
        & \vdash P \{ a/z,b/y' \} \typInf \Gamma, a:\ol{A}, b:\ol{B}
    \end{align*}

    \caption{
        Example of subject reduction (cf.\ \Cref{t:APCP:sr}) in Rule~\protect\ruleLabel{red-send-recv}.
        The well-typedness of the process before reduction allows us to infer its typing derivation, also giving us the typing of $P$.
        Typing the process after reduction is then a matter of inductively substituting names in the typing derivation of $P$.
    }
    \label{f:APCP:tp}
\end{figure}

Now we can prove that reduction preserves typing, i.e., \emph{subject reduction}:
\fi

\ifmpstAPCP
\begin{theorem}[Type Preservation]
    \label{t:mpstAPCP:APCP:tp}
\else
\begin{theorem}[Subject Reduction]
    \label{t:APCP:sr}
\fi
    If $\Omega \vdash P \typInf \Gamma$ and $P \redd Q$, then there exists $\Gamma'$ such that $\Omega \vdash Q \typInf \Gamma'$ and ${(\Omega \vdash P \typInf \Gamma) \unfrel* (\Omega \vdash Q \typInf \Gamma')}$.
\end{theorem}

\ifmpstAPCP\else
\begin{sketch}
    By induction on the derivation of $P \redd Q$ (\ifthesis\Cref{d:APCP:APCP:redd}\else\Cref{f:procdef} (bottom)\fi), we find a $\Gamma'$ such that $\Omega \vdash Q \typInf \Gamma'$ and $(\Omega \vdash P \typInf \Gamma) \unfrel* (\Omega \vdash Q \typInf \Gamma')$.

    The cases of Rules~\ruleLabel{red-res} and~\ruleLabel{red-par} follow directly from the IH.
    Consider, e.g., $\pRes{xy} ( P \| R ) \redd \pRes{xy} ( Q \| R )$ derived from $P \redd Q$.
    By inversion of typing on the process before reduction, e.g., $\Omega \vdash P \typInf \Gamma , x:A$.
    Then, by the IH, $(\Omega \vdash P \typInf \Gamma, x:A) \unfrel* (\Omega \vdash Q \typInf \Gamma', x:A')$.
    The thesis then follows from an application of Rules~\ruleLabel{unf-par} and~\ruleLabel{unf-res}.

    The case of Rule~\ruleLabel{red-cong} follows from the IH and subject congruence (\Cref{t:APCP:sc}).
    To be precise, the rule says that (i)~$P \equiv P'$, (ii)~$P' \redd Q'$, and (iii)~$Q' \equiv Q$ imply $P \redd Q$.
    By \Cref{t:APCP:sc} (subject congruence) on assumptions~(i) and~(iii), we have $(\Omega \vdash P \typInf \Gamma) \unfrel* (\Omega \vdash P' \typInf \Gamma')$ and $(\Omega' \vdash Q' \typInf \Delta') \unfrel* (\Omega' \vdash Q \typInf \Delta)$.
    Also, by the IH on assumption~(ii), $(\Omega \vdash P' \typInf \Gamma') \unfrel* (\Omega' \vdash Q' \typInf \Delta')$ where $\Omega = \Omega'$.
    Then, by \Cref{p:APCP:unfrelTrans} (transitivity of $\unfrel*$), $(\Omega \vdash P \typInf \Gamma) \unfrel* (\Omega \vdash Q \typInf \Delta)$, proving the thesis.

    Key cases are Rules~\ruleLabel{red-send-recv}, \ruleLabel{red-sel-bra}, and~\ruleLabel{red-fwd}.
    \Cref{f:APCP:tp} shows the representative instance of Rule~\ruleLabel{red-send-recv}, and an example where Rules~\ruleLabel{red-res}, \ruleLabel{red-par}, and~\ruleLabel{red-cong} are used.
\end{sketch}

Protocol fidelity ensures that processes respect their intended (session) protocols.
Communication safety ensures the absence of communication errors and mismatches in processes.
Correct typability gives a static guarantee that a process conforms to its ascribed session protocols; type preservation gives a dynamic guarantee.
Because session types describe the intended protocols and error-free exchanges, type preservation entails both protocol fidelity and communication safety.
For a detailed account, we refer the curious reader to the early work by Honda \etal~\cite{conf/esop/HondaVK98}, which defines error processes
and shows by contradiction that well-typed processes do not reduce to an error.

In what follows, we consider a process to be deadlocked if it is not the inactive process and cannot reduce.
Our deadlock-freedom result for \APCP adapts that for \PCP~\cite{conf/fossacs/DardhaG18}.
The equivalent of our Rule~\ruleLabel{typ-res} in \PCP is Rule~\ruleLabel{cycle}.
Deadlock-freedom for \PCP involves three steps to eliminate applications of Rule~\ruleLabel{cycle}:
\begin{enumerate}

    \item
        First, \ruleLabel{cycle}-elimination states that we can remove all applications of \ruleLabel{cycle} in a typing derivation without affecting the derivation's assumptions and conclusion.

    \item
        Only the removal of \emph{top-level} applications of \ruleLabel{cycle} captures the intended process semantics; the removal of other applications of \ruleLabel{cycle} corresponds to reductions behind prefixes, which is not allowed~\cite{conf/icfp/Wadler12,conf/fossacs/DardhaG18}.
        Therefore, the second step is \emph{top-level deadlock-freedom} (referred to here as \emph{progress}), which states that a process with a top-level application of \ruleLabel{cycle} reduces until there are no top-level applications of \ruleLabel{cycle} left.
        This steps requires commuting conversions: a process with a top-level application of \ruleLabel{cycle} may not have reductions ready, so commuting conversions are used to remove top-level applications of \ruleLabel{cycle} by blocking them with communications.

    \item
        Third, deadlock-freedom follows for processes typable under empty contexts.

\end{enumerate}
Here, we adapt and address \ruleLabel{typ-res}-elimination and progress in one proof.

As mentioned before, binding \APCP's asynchronous sends and selections to continuations involves additional, low-level uses of \ruleLabel{typ-res}, which we cannot eliminate through process reduction.
Therefore, we establish progress for \emph{live processes} (\Cref{t:APCP:progress}).
A process is live if it is equivalent to a restriction on \emph{active names} that perform unguarded communications.
This way, e.g., in $\pOut x[y,z]$ the name $x$ is active, but $y$ and $z$ are not.
We additionally need a notion of \emph{evaluation context}, under which reducible forwarders may occur.
\fi
\fi

\ifmpstAPCP\else
\def\dAPCPan{
    \begin{definition}[Active Names]
        \ifLASTn \label{d:LASTn:APCP:an}
        \else \label{d:APCP:an} \fi
        The \emph{set of active names} of $P$, denoted $\an(P)$, contains the (free) names that are used for non-blocked communications (send, receive, selection, branch):
        \ifLASTn
        \begin{align*}
            \an(\pOut x[y,z])
            &\deq
            \{x\}
            &
            \an(\pIn x(y,z) ; P)
            &\deq
            \{x\}
            &
            \an(\0)
            &\deq
            \emptyset
            \\
            \an(\pSel x[z] < \ell)
            &\deq
            \{x\}
            &
            \an(\pBra x(z) > {\{i: P_i\}_{i \in I}})
            &\deq
            \{x\}
            &
            \an(\pFwd [x<>y])
            &\deq
            \emptyset
            \\
            \an(P \| Q)
            &\deq
            \an(P) \cup \an(Q)
            \\
            \an(\pRes{xy} P)
            &\deq
            \an(P) \setminus \{x,y\}
        \end{align*}
        \else
        \begin{align*}
            \an(\pOut x[y,z])
            &\deq
            \{x\}
            &
            \an(\pIn x(y,z) ; P)
            &\deq
            \{x\}
            &
            \an(\0)
            &\deq
            \emptyset
            \\
            \an(\pSel x[z] < \ell)
            &\deq
            \{x\}
            &
            \an(\pBra x(z) > {\{i: P_i\}_{i \in I}})
            &\deq
            \{x\}
            &
            \an(\pFwd [x<>y])
            &\deq
            \emptyset
            \\
            \an(P \| Q)
            &\deq
            \an(P) \cup \an(Q)
            &
            \an(\pRec X(\tilde{x}) ; P)
            &\deq
            \an(P)
            \\
            \an(\pRes{xy} P)
            &\deq
            \an(P) \setminus \{x,y\}
            &
            \an(\pCall X<\tilde{x}>)
            &\deq
            \emptyset
        \end{align*}
        \fi
    \end{definition}
}

\ifLASTn\else
\dAPCPan
\fi

\ifLASTn
The following notion of evaluation context will be useful in \Cref{s:LASTn:LASTn:trans}.
\fi
\begin{definition}[Evaluation Context]
    \ifLASTn \label{d:LASTn:APCP:evalCtx}
    \else \ifmpstAPCP \label{d:mpstAPCP:APCP:evalCtx}
    \else \label{d:APCP:evalCtx} \fi \fi
    Evaluation contexts ($\evalCtx{E}$) have the following grammar:
    \[
        \evalCtx{E} ::= \evalHole \sepr \evalCtx{E} \| P \sepr \pRes{xy} \evalCtx{E} \ifLASTn\else \sepr \pRec X(\tilde{x}) ; \evalCtx{E} \fi
    \]
    We write $\evalCtx{E}[P]$ to denote the process obtained by replacing the hole $\evalHole$ in $\evalCtx{E}$ by $P$.
\end{definition}

\def\dAPCPLive{
    \begin{definition}[Live Process]
        \ifLASTn \label{d:LASTn:APCP:live}
        \else \label{d:APCP:live} \fi
        A process $P$ is \emph{live}, denoted $\live(P)$, if
        \begin{enumerate}

            \item
                there are names $x,y$ and process $P'$ such that $P \equiv \pRes{xy} P'$ with $x,y \in \an(P')$, or

            \item
                there are names $x,y,z$ and process $P'$ such that $P \equiv \evalCtx[\big]{E}[\pRes{yz} ( \pFwd [x<>y] \| P' )]$ and $z \neq x$ (i.e., the forwarder is independent).

        \end{enumerate}
    \end{definition}
}

\ifLASTn\else
\dAPCPLive

We additionally need to account for recursion: as recursive definitions do not directly entail reductions, we must fully unfold them before eliminating applications of \ruleLabel{typ-res}:

\begin{restatable}[Unfolding]{lemma}{lAPCPUnfold}
    \label{l:APCP:unfold}
    If $\Omega \vdash P \typInf \Gamma$, then there is a process $P^\star$ such that $P^\star \equiv P$ and $P^\star$ is not of the form $\pRec X(\tilde{z}) ; Q$.
\end{restatable}

\ifthesis\else
\begin{proof}
    By induction on the number $n$ of consecutive recursive definitions prefixing $P$, such that $P$ is of the form $\pRec X_1(\tilde{z}) ; \ldots ; \pRec X_n(\tilde{z}) ; Q$.
    If $n = 0$, the thesis follows immediately, as $P \equiv P$.
    Otherwise, $n \geq 1$.
    Then there are $X,Q$ such that $P = \pRec X(\tilde{z}) ; Q$, where $Q$ starts with $n-1$ consecutive recursive definitions.
    Let $R \deq Q \big\{ \big( \pRec X(\tilde{y}) ; Q \{ \tilde{y}/\tilde{z} \} \big) / \pCall X<\tilde{y}> \big\}$.
    Clearly, $R \equiv P$.
    Then, because $R$ starts with $n-1$ consecutive recursive definitions, the thesis follows by appealing to the IH.
\end{proof}
\fi

Dardha and Gay's progress result concerns a sequence of reduction steps that reaches a process that is not live anymore~\cite{conf/fossacs/DardhaG18}.
In our case, progress concerns a single reduction step only, because recursive processes might stay live across reductions forever.
Moreover, because of our definition of liveness, we do not need commuting conversions for this step.

\begin{restatable}[Progress]{theorem}{tAPCPProgress}
    \label{t:APCP:progress}
    If $\emptyset \vdash P \typInf \Gamma$ and $\live(P)$, then there is a process $Q$ such that $P \redd Q$.
\end{restatable}

\ifthesis
\begin{sketch}
    After unfolding $P$, from Liveness (\Cref{d:APCP:live}) we can infer the shape of~$P$.
    A reduction from \Cref{d:APCP:APCP:redd} then follows straightforwardly.
\end{sketch}
\else
\begin{proof}
    We distinguish the two cases of $\live(P)$: $P$ contains a restriction on a pair of active names, or $P$ contains a restriction on a forwarded name under a reduction context.
    In both cases, we first unfold any recursive definitions preceding the invpolved prefixes/forwarders, resulting in $P^\star \equiv P$.
    By subject congruence (\Cref{t:APCP:sc}), there exists $\Gamma'$ such that \mbox{$\emptyset \vdash P^\star \typInf \Gamma'$} and $(\vdash P^\star \typInf \Gamma') \unfrel \Gamma$.
    By \Cref{d:APCP:unfrel}, the only difference between $\Gamma'$ and $\Gamma$ is the unfolding of recursive types.
    \begin{enumerate}

        \item
            $P^\star$ contains a restriction on a pair of active names.
            That is, $P^\star \equiv \pRes{xy} {P^\star}'$ and \mbox{$x,y \in \an({P^\star}')$}.
            The rest of the analysis depends on how $x$ and $y$ occur as active names in ${P^\star}'$.
            As a representative case, we consider that $x$ occurs as the subject of a send.
            By inversion of typing, then $y$ occurs as the subject of a receive.
            Hence, we have $P^\star \equiv \evalCtx[\big]{E}[\pRes{xy} ( \pOut x[a,b] \| \pIn y(v,z) ; {P^\star}'' )]$.
            Let $Q \deq \evalCtx{E}[{P^\star}'' \{ a/v,b/z \}]$.
            Then, by Rule~\ruleLabel{red-send-recv}, $P^\star \redd Q$.
            Hence, by Rule~\ruleLabel{red-cong}, $P \redd Q$, proving the thesis.

        \item
            $P^\star$ contains a restriction on a forwarded name under a reduction context.
            That is, $P^\star \equiv \evalCtx[\big]{E}[\pRes{yz} ( \pFwd [x<>y] \| {P^\star}' )]$ where $z \neq x$.
            Let $Q = \evalCtx{E}[{P^\star}' \{ x/z \}]$.
            Then, by Rule~\ruleLabel{red-fwd}, $P^\star \redd Q$.
            Hence, by Rule~\ruleLabel{red-cong}, $P \redd Q$, proving the thesis.
            \qedhere

    \end{enumerate}
\end{proof}
\fi
\fi
\fi

Our deadlock-freedom result concerns processes typable under empty contexts\ifmpstAPCP:\else~(as in, e.g., Caires and Pfenning~\cite{conf/concur/CairesP10} and Dardha and Gay~\cite{conf/fossacs/DardhaG18}). \fi
\ifLASTn\else\ifmpstAPCP\else
We first need a lemma which ensures that non-live processes typable under empty contexts do not contain prefixes (sends/receives/selections/branches) or independent forwarders (whose endpoints are not bound together using restriction).
This lemma is in fact the crux of our deadlock-freedom result, as it relies on the priority checks induced by our typing system:

\begin{restatable}{lemma}{lAPCPNotLiveNoAction}
    \label{l:APCP:notLiveNoAction}
    If $\emptyset \vdash P \typInf \emptyset$ and $P$ is not live, then $P$ has no prefixes or independent forwarders.
\end{restatable}

\ifthesis
\begin{sketch}
    By contradiction.
    We consider $P \equiv P^\star = \pRes{x_i x'_i}_{i \in I} \prod_{j \in J} P_j$ where, for every~\mbox{$j \in J$}, $P_j$ is a prefix or forwarder (i.e., every $P_j$ can be seen as an atomic \emph{thread}). 
    We then apply induction on the size of $J$.
    In the base case ($J = \emptyset$), the contradiction is immediate.
    In the inductive case, we pick $P_j$ such that its foremost prefix/forwarder has the least priority of all threads.
    We then consider all the possible foremost prefixes/forwarders of $P_j$, and in each case rely on the assumption that $P$ is not live to derive contradictory priority requirements.
\end{sketch}
\else
\begin{proof}
    W.l.o.g., assume all recursion in $P$ is unfolded.
    Take $P^\star \equiv P$ such that \mbox{$P^\star = \pRes{x_i x'_i}_{i \in I} \prod_{j \in J} P_j$} where, for every $j \in J$, $P_j$ is a prefix or forwarder.
    That is, every $P_j$ is a \emph{thread} which cannot be broken down further into parallel components or restrictions.
    
    By abuse of notation, we write $\pr(P_j)$ to denote the priority of the type of the subject of the prefix of $P_j$.
    Also, we write $\pr(x)$ to denote the priority of the type of $x$.

    Toward a contradiction, assume that there is at least one prefix or independent forwarder in $P$.
    We apply induction on the size of $J$.
    \begin{itemize}

        \item
            In the base case, $J = \emptyset$.
            Then there cannot be any prefixes or independent forwarders in $P$: a contradiction.
            
        \item
            In the inductive case, $J = J' \cup \{j\}$.
            W.l.o.g., assume that, for every $j' \in J'$, $\pr(P_j) \leq \pr(P_{j'})$ (i.e., pick $P_j$ as one of the prefixes/forwarders with the least priority of all threads).
            By assumption, $P_j$ denotes a prefix or forwarder on some endpoint $x_i$.
            By well-typedness, $x_i$ is connected through restriction to $x'_i$.
            By Rule~\ruleLabel{typ-res}, $\pr(x_i) = \pr(x'_i)$.
            Further analysis depends on $P_j$: a receive (or, analogously, a branch), a send (or, analogously, a selection), or a forwarder.
            \begin{itemize}
                    
                \item
                    Suppose $P_j = \pIn x_i(y,z) ; P'_j$.
                    By Rule~\ruleLabel{typ-recv}, \mbox{$x'_i \notin \fn(P'_j)$}, because the rule requires $\pr(x_i) < \pr(x'_i)$: otherwise, this would contradict the fact that $\pr(x_i) = \pr(x'_i)$.
                    Hence, there exists $j' \in J'$ such that $x'_i \in \fn(P_{j'})$.

                    If $J' = \emptyset$, the contradiction is immediate.
                    Otherwise, the analysis depends on the prefix of $P_{j'}$.
                    Since $P$ is not live, this cannot be a send with subject $x'_i$ or an independent forwarder on $x'_i$.
                    The prefix can also not be a dependent forwarder, for $x_i$ already appears in $P_j$.
                    This leaves us with two possibilities: a receive (or, analogously, a branch) on another endpoint, or a send (or, analogously, a selection) with object $x'_i$.
                    \begin{itemize}

                        \item
                            If $P_{j'} = \pIn x_k(v,w) ; P'_{j'}$ where $x'_i \in \fn(P'_{j'})$, then, by Rule~\ruleLabel{typ-recv}, $\pr(x_k) < \pr(x'_i)$.
                            Hence, $\pr(P_{j'}) = \pr(x_k) < \pr(x'_i) = \pr(x_i) = \pr(P_j)$: this contradicts the assumption that $\pr(P_j) \leq \pr(P_{j'})$.

                        \item
                            If $P_{j'} = \pOut x[a,b]$ where $x'_i \in \{a,b\}$, w.l.o.g., assume $x'_i = a$.
                            Then, by Rule~\ruleLabel{typ-send}, $\pr(x_k) < \pr(x'_i)$.
                            The contradiction follows as in the previous case.

                    \end{itemize}
                    
                \item
                    Suppose $P_j = \pOut x_i[a,b]$.
                    By Rule~\ruleLabel{typ-send}, $x'_i \notin \{a,b\}$, because the rule requires $\pr(x_i) < \pr(x'_i)$: otherwise, this would contradict the fact that $\pr(x_i) = \pr(x'_i)$.
                    Hence, there exists $j' \in J'$ such that $x'_i \in \fn(P_{j'})$.
                    The contradiction follows as in the previous case.
                    
                \item
                    Suppose $P_j = \pFwd [x_i<>x_k]$.
                    By well-typedness, $x_i$ is connected through restriction to~$x'_i$.
                    Since $P$ is not live, it must be that $x'_i = x_k$: $P_j$ is not an independent forwarder.
                    Then there must be a $j' \in J'$ where $P_{j'}$ fulfills the assumption that $P$ contains at least one prefix or independent forwarder.
                    The contradiction then follows from the IH.
                    \qedhere

            \end{itemize}

    \end{itemize}
\end{proof}
\fi

We now state our deadlock-freedom result:
\fi
\fi

\ifLASTn
\begin{theorem}[Deadlock-freedom for \APCP]
    \label{t:LASTn:APCP:df}
    If $\vdash P \typInf \emptyset$, then either $P \equiv \0$ or $P \redd Q\mkern2mu$ for some $Q$.
\end{theorem}
\else
\ifmpstAPCP
\begin{theorem}[Deadlock-freedom for \APCP]
    \label{t:mpstAPCP:APCP:df}
    If $\emptyset \vdash P \typInf \emptyset$, then either $P \equiv \0$ or $P \redd Q\mkern2mu$ for some $Q$.
\end{theorem}
\else
\begin{restatable}[Deadlock-freedom]{theorem}{tAPCPDF}
    \label{t:APCP:df}
    If $\emptyset \vdash P \typInf \emptyset$, then either $P \equiv \0$ or $P \redd Q\mkern2mu$ for some $Q$.
\end{restatable}
\fi
\fi

\ifLASTn\else\ifmpstAPCP\else
\ifthesis
\begin{sketch}
    If $P$ is live, from progress (\Cref{t:APCP:progress}) it follows that $P \redd Q$.
    Otherwise, it follows from \Cref{l:APCP:notLiveNoAction} that $P \equiv \0$.
\end{sketch}
\else
\begin{proof}
    The analysis depends on whether $P$ is live or not.
    \begin{itemize}
        \item
            If $P$ is not live, then, by \Cref{l:APCP:notLiveNoAction}, it does not contain any prefixes or independent forwarders.
            Any recursive definitions in $P$ are thus of the form $\pRec X_1() ; \ldots ; \pRec X_n() ; \0$: contractiveness requires recursive calls to be prefixed by receives/branches or bound to parallel sends/selections/forwarders, of which there are none.
            Hence, we can use structural congruence to rewrite each recursive definition in $P$ to $\0$ by unfolding, yielding $P' \equiv P$.
            Any dependent forwarders in $P$ are of the form $\pRes{xy} \pFwd [x<>y]$, and can be rewritten to $\0$ using Rule~\ruleLabel{cong-res-fwd} (cf. \Cref{d:APCP:APCP:redd}).
            The remaining derivation of $P'$ only contains applications of Rule~\ruleLabel{typ-inact}, \ruleLabel{typ-par}, \ruleLabel{typ-end}, or~\ruleLabel{typ-res} on closed names.
            It follows straightforwardly that $P \equiv P' \equiv \0$.

        \item
            If $P$ is live, by \Cref{t:APCP:progress} there is $Q$ s.t.\ $P \redd Q$.
            \qedhere
    \end{itemize}
\end{proof}
\fi
\fi
\fi

\subsection{Reactivity}
\label{s:APCP:APCP:react}

Processes typable under empty contexts are not only deadlock-free, they are \emph{reactive}, in the following sense: for each name in the process, we can eventually observe a reduction involving that name.
To formalize this property, we define \emph{labeled reductions}, which expose details about communications:

\begin{definition}[Labeled Reductions]
    \ifmpstAPCP \label{d:mpstAPCP:APCP:procLred}
    \else \ifAPCP \label{d:APCP:procLred}
    \else \label{d:procLred} \fi \fi
    Consider the labels
    \begin{align*}
        \alpha ::= \pFwd [x<>y] \sepr \rLbl x>y:a \sepr \rLbl x>y:\ell
        \qquad\qquad
        \text{(forwarding, send/receive, selection/branching)}
    \end{align*}
    where each label has subjects $x$ and $y$.
    The following rules define \emph{labeled reductions} $P \lredd{\alpha} Q$:
    \begin{mathpar}
        \begin{bussproof}[lred-send-recv]
            \bussAx{
                \pRes{xy} ( \pOut x[a,b] \| \pIn y(v,z) ; P )
                \lredd{\rLbl x>y:a}
                P \{ a/v,b/z \}
            }
        \end{bussproof}
        \and
        \begin{bussproof}[lred-sel-bra]
            \bussAssume{
                j \in I
            }
            \bussUn{
                \pRes{xy} ( \pSel x[b] < j \| \pBra y(z) > {\{i: P_i\}_{i \in I}} )
                \lredd{\rLbl x>y:j}
                P_j \{ b/z \}
            }
        \end{bussproof}
        \and
        \begin{bussproof}[lred-fwd]
            \bussAx{
                \pRes{yz} ( \pFwd [x<>y] \| P )
                \lredd{\pFwd [x<>y]}
                P \{ x/z \}
            }
        \end{bussproof}
        \and
        \begin{bussproof}[lred-cong]
            \bussAssume{
                P \equiv P'
            }
            \bussAssume{
                P' \lredd{\alpha} Q'
            }
            \bussAssume{
                Q' \equiv Q
            }
            \bussTern{
                P \lredd{\alpha} Q
            }
        \end{bussproof}
        \and
        \begin{bussproof}[lred-res]
            \bussAssume{
                P \lredd{\alpha} Q
            }
            \bussUn{
                \pRes{xy} P \lredd{\alpha} \pRes{xy} Q
            }
        \end{bussproof}
        \and
        \begin{bussproof}[lred-par]
            \bussAssume{
                P \lredd{\alpha} Q
            }
            \bussUn{
                P \| R \lredd{\alpha} Q \| R
            }
        \end{bussproof}
    \end{mathpar}
\end{definition}

\def\makePLredd{
    For any $P$ and $P'$, $P \redd P'$ if and only if there exists a label $\alpha$ such that~\mbox{$P \lredd{\alpha} P'$}.
}
\ifmpstAPCP
\begin{proposition}
    \makePLredd
\end{proposition}
\else
\begin{restatable}{proposition}{pLredd}
    \makePLredd
\end{restatable}
\fi

\ifthesis\else
\begin{proof}
    Immediate by definition, because each reduction in \ifthesis\Cref{d:APCP:APCP:redd} \else\Cref{f:procdef} (bottom) \fi corresponds to a labeled reduction, and vice versa.
\end{proof}
\fi

Our reactivity result states that all processes typable under empty contexts have a finite reduction sequence (denoted $\reddF$) that enables a labeled reduction involving a \emph{pending} name---it is the subject of a communication and is not bound by an input (see below).
Clearly, the typed process may have other reduction sequences, not necessarily finite.

\begin{definition}[Pending Names]
    \ifmpstAPCP \label{d:mpstAPCP:APCP:pn}
    \else \label{d:APCP:pn} \fi
    Given a process $P$, we define the set of \emph{pending names} of~$P$, denoted $\pn(P)$, as follows:
    \begin{align*}
        \pn(\pOut x[y,z])
        &\deq
        \{x\}
        &
        \pn(\pIn x(y,z) ; P)
        &\deq
        \{x\} \cup (\pn(P) \setminus \{y,z\})
        \span\span
        \\
        \pn(\pSel x[z] < \ell)
        &\deq
        \{x\}
        &
        \pn(\pBra x(z) > {\{i: P_i\}_{i \in I}})
        &\deq
        \{x\} \cup (\bigcup_{i \in I} \pn(P_i) \setminus \{z\})
        \span\span
        \\
        \pn(P \| Q)
        &\deq
        \pn(P) \cup \pn(Q)
        &
        \pn(\pRec X(\tilde{x}) ; P)
        &\deq
        \pn(P)
        &
        \pn(\0)
        &\deq
        \emptyset
        \\
        \pn(\pRes{xy} P)
        &\deq
        \pn(P)
        &
        \pn(\pCall X<\tilde{x}>)
        &\deq
        \emptyset
        &
        \pn(\pFwd [x<>y])
        &\deq
        \{x,y\}
    \end{align*}
\end{definition}

\ifmpstAPCP\else
Note that the proof of reactivity below does not rely on deadlock-freedom: suppose we are observing a blocked pending communication in a process with a parallel recursive definition; deadlock-freedom ensures a reduction from the recursive definition which would not unblock the pending communication we are observing.
Instead, the proof relies on a  priority analysis (similar to the one in the proof of \Cref{l:APCP:notLiveNoAction}) to unblock pending communications.
\fi

\def\makeTAPCPReact{
    Suppose given a process $\emptyset \vdash P \typInf \emptyset$.
    Then, for every $x \in \pn(P)$ there exists a process $P'$ such that $P \reddF P'$ and $P' \lredd{\alpha} Q$, for some process $Q$ and label $\alpha$ with subject $x$.
}
\ifmpstAPCP
\begin{theorem}[Reactivity]
    \label{t:mpstAPCP:APCP:react}
    \makeTAPCPReact
\end{theorem}
\else
\begin{restatable}[Reactivity]{theorem}{tAPCPReact}
    \label{t:APCP:react}
    \makeTAPCPReact
\end{restatable}
\fi

\ifthesis
\ifmpstAPCP\else
\begin{sketch}
    The prefix on $x$ is connected to some parallel dual prefix on $y$ in $P$.
    Similar to the reasoning in the proof of \Cref{l:APCP:notLiveNoAction}, we can find a finite chain of inputs blocking this communication on $y$.
    We then inductively reduce $P$ to eventually unblock the dual prefix on~$y$ and show the required labeled reduction.
\end{sketch}
\fi
\else
\begin{proof}
    Take any $x \in \pn(P)$.
    Because $P$ is typable under empty contexts, $x$ is bound to some $y \in \pn(P)$ by restriction.
    By typing, in $P$ there is exactly one prefix on $x$ and one prefix on $y$ (they may also occur in forwarder processes).
    Following the restrictions on priorities in the typing of $x$ and $y$ in $P$, the prefixes on $x$ and $y$ cannot occur sequentially in $P$ (cf.\ the proof of \Cref{l:APCP:notLiveNoAction} for details on this reasoning).
    By typability, the prefix on $y$ is dual to the prefix on $x$.

    We apply induction on the number of receives, branches, and recursive definitions in $P$ blocking the prefixes on $x$ and $y$, denoted $n$ and $m$, respectively.
    Because $P$ is typable under empty contexts, the blocking receives and branches that are on names in $\pn(P)$ also have to be bound to pending names by restriction.
    The prefixes on these connected names may also be prefixed by receives, branches, and recursive definitions, so we may need to unblock those prefixes as well.
    Since there can only be a finite number of names in any given process, we also apply induction on the number of prefixes blocking these connected prefixes.
    \begin{itemize}

        \item
            If $n = 0$ and $m = 0$, then the prefixes on $x$ and $y$ occur at the top-level; because they do not occur sequentially, the communication between $x$ and $y$ can take place immediately.
            Hence, $P \lredd{\alpha} Q$ where $x$ and $y$ are the subjects of $\alpha$.
            This proves the thesis, with $P' = P$.

        \item
            If $n > 0$ or $m > 0$, the analysis depends on the foremost prefixes blocking the prefixes on $x$ and $y$.

            If either of these blocking prefixes is a recursive definition ($\pRec X(\tilde{y})$), we unfold the recursion.
            Because a corresponding recursive call ($\pCall X<\tilde{z}>$) cannot occur as a prefix, the effect of unfolding either (i)~triggers prefixes that occur in parallel to those on $x$ and~$y$, or (ii)~the prefixes on $x$ or $y$ precede the punfolded recursive call.
            In either case, the number of prefixes decreases, and the thesis follows from the IH.

            Otherwise, if neither foremost prefix is a recursive definition, then the foremost prefixes must be on names in $\pn(P)$.
            Consider the prefix that is typable with the least priority.
            W.l.o.g.\ assume that this is the foremost prefix of $x$.
            Suppose this prefix is on some name~$w$ connected to another name $z \in \pn(P)$ by restriction.
            By typability, the priority of $w$ is less than that of $x$ and all of the prefixes in between.
            This means that the number of prefixes blocking the prefix on $z$ strictly decreases.
            Hence, by the IH, $P \reddF P'' \lredd{\alpha'} Q'$ in a finite number of steps, where $w$ and $z$ are the subjects of $\alpha'$.
            The communication between $w$ and $z$ can be performed, and $n$ decreases.
            By type preservation (\Cref{t:APCP:sr}), $\emptyset \vdash Q' \typInf \emptyset$.
            The thesis then follows from the IH: $P \reddF P'' \lredd{\alpha'} Q' \reddF P' \lredd{\alpha} Q$ in finite steps, where $x$ and $y$ are the subjects of $\alpha$.
            \qedhere

    \end{itemize}
\end{proof}
\fi

\subsection{Typing Milner's  Cyclic Scheduler}
\label{s:APCP:APCP:milner}

\begin{figure}[t]
    \[
        \begin{bussproof}
            \bussAx[\ruleLabel{typ-var}]{
                X: (T_{\aChan{a_1}}, T_{\cChan{c_n}}, T_{\dChan{d_1}}) \vdash \begin{array}[t]{@{}l@{}}
                    \aChan{a_1}: \mu X . (\lift{t_1} T_{\aChan{a_1}}),
                    \cChan{c_n}: \mu X . (\lift{t_1} T_{\cChan{c_n}}),
                    \\
                    \dChan{d_1}: \mu X . (\lift{t_1} T_{\dChan{d_1}})
                \end{array}
                \quad
                \eqref{eq:APCP:pri:bra1}
            }
            \bussUn[\ruleLabel{typ-bra}]{
                X: T_X \vdash \begin{array}[t]{@{}l@{}}
                    \aChan{a_1}: \mu X . (\lift{t_1} T_{\aChan{a_1}}),
                    \cChan{c_n}: \&^{\rho_n} \{ \sff{next} : \mu X . (\lift{t_1} T_{\cChan{c_n}}) \},
                    \\
                    \dChan{d_1}: \mu X . (\lift{t_1} T_{\dChan{d_1}})
                \end{array}
                \quad
                \eqref{eq:APCP:pri:bra2}
            }
            \bussUn[\ruleLabel{typ-bra}]{
                X: T_X \vdash \begin{array}[t]{@{}l@{}}
                    \aChan{a_1}: \mu X . (\lift{t_1} T_{\aChan{a_1}}),
                    \cChan{c_n}: \&^{\pi_n} \{ \sff{start} : \&^{\rho_n} \{ \sff{next} : \mu X . (\lift{t_1} T_{\cChan{c_n}}) \} \},
                    \\
                    \dChan{d_1}: \mu X . (\lift{t_1} T_{\dChan{d_1}})
                \end{array}
                \quad
                \eqref{eq:APCP:pri:sel1}
            }
            \bussUn[\ruleLabel{typ-sel$\ast$}]{
                X: T_X \vdash
                \aChan{a_1}: \mu X . (\lift{t_1} T_{\aChan{a_1}}), \begin{array}[t]{@{}l@{}}
                    \cChan{c_n}: \&^{\pi_n} \{ \sff{start} : \&^{\rho_n} \{ \sff{next} : \mu X . (\lift{t_1} T_{\cChan{c_n}}) \} \},
                    \\
                    \dChan{d_1}: \oplus^{\rho_1} \{ \sff{next} : \mu X . (\lift{t_1} T_{\dChan{d_1}}) \}
                \end{array}
                \quad
                \eqref{eq:APCP:pri:bra3}
            }
            \bussUn[\ruleLabel{typ-bra}]{
                X: T_X \vdash \begin{array}[t]{@{}l@{}}
                    \aChan{a_1}: \&^{\kappa_1} \{ \sff{ack} : \mu X . (\lift{t_1} T_{\aChan{a_1}}) \},
                    \\
                    \cChan{c_n}: \&^{\pi_n} \{ \sff{start} : \&^{\rho_n} \{ \sff{next} : \mu X . (\lift{t_1} T_{\cChan{c_n}}) \} \},
                    \\
                    \dChan{d_1}: \oplus^{\rho_1} \{ \sff{next} : \mu X . (\lift{t_1} T_{\dChan{d_1}}) \}
                \end{array}
                \quad
                \eqref{eq:APCP:pri:sel2}
            }
            \bussUn[\ruleLabel{typ-sel$\ast$}]{
                X: T_X \vdash \begin{array}[t]{@{}l@{}}
                    \aChan{a_1}: \oplus^{\pri_1} \{ \sff{start} : \&^{\kappa_1} \{ \sff{ack} : \mu X . (\lift{t_1} T_{\aChan{a_1}}) \} \},
                    \\
                    \cChan{c_n}: \&^{\pi_n} \{ \sff{start} : \&^{\rho_n} \{ \sff{next} : \mu X . (\lift{t_1} T_{\cChan{c_n}}) \} \},
                    \\
                    \dChan{d_1}: \oplus^{\rho_1} \{ \sff{next} : \mu X . (\lift{t_1} T_{\dChan{d_1}}) \}
                \end{array}
                \quad
                \eqref{eq:APCP:pri:sel3}
            }
            \bussUn[\ruleLabel{typ-sel$\ast$}]{
                X: \underbrace{
                    (T_{\aChan{a_1}}, T_{\cChan{c_n}}, T_{\dChan{d_1}})
                }_{T_X} \vdash \begin{array}[t]{@{}l@{}}
                    \aChan{a_1}: \oplus^{\pri_1} \{ \sff{start} : \&^{\kappa_1} \{ \sff{ack} : \mu X . (\lift{t_1} T_{\aChan{a_1}}) \} \},
                    \\
                    \cChan{c_n}: \&^{\pi_n} \{ \sff{start} : \&^{\rho_n} \{ \sff{next} : \mu X . (\lift{t_1} T_{\cChan{c_n}}) \} \},
                    \\
                    \dChan{d_1}: \oplus^{\pi_1} \{ \sff{start} : \oplus^{\rho_1} \{ \sff{next} : \mu X . (\lift{t_1} T_{\dChan{d_1}}) \} \}
                \end{array}
                \quad
                \eqref{eq:APCP:pri:rec}
            }
            \bussUn[\ruleLabel{typ-rec}]{
                \emptyset \vdash \begin{array}[t]{@{}l@{}}
                    \aChan{a_1}: \mu X . \underbrace{
                        \oplus^{\pri_1} \{ \sff{start} : \&^{\kappa_1} \{ \sff{ack} : X \} \}
                    }_{T_{\aChan{a_1}}},
                    \cChan{c_n}: \mu X . \underbrace{
                        \&^{\pi_n} \{ \sff{start} : \&^{\rho_n} \{ \sff{next} : X \} \}
                    }_{T_{\cChan{c_n}}},
                    \\
                    \dChan{d_1}: \mu X . \underbrace{
                        \oplus^{\pi_1} \{ \sff{start} : \oplus^{\rho_1} \{ \sff{next} : X \} \}
                    }_{T_{\dChan{d_1}}}
                \end{array}
            }
        \end{bussproof}
    \]

    \caption{Typing derivation of the leader scheduler $A_1$ of Milner's cyclic scheduler (processes omitted).}
    \label{f:APCP:milnerLeader}
\end{figure}

Here we show that our specification of Milner's cyclic scheduler from \Cref{s:APCP:milner} is typable in \APCP, and thus deadlock-free (cf.\ \Cref{t:APCP:df}).
Let us recall the process definitions of the leader and followers, omitting braces `$\{\ldots\}$' for branches with one option:
\cyclicAsyncMilnerScheds
\Cref{f:APCP:milnerLeader} gives the typing derivation of $A_1$, omitting processes from judgments, with the following priority requirements:
\begin{align}
    t_1 &> \max(\pri_1,\kappa_1,\pi_n,\rho_n,\pi_1,\rho_1) \label{eq:APCP:pri:rec}
    \\
    \pi_1 &< \rho_1 \label{eq:APCP:pri:sel3}
    \\
    \pri_1 &< \kappa_1 \label{eq:APCP:pri:sel2}
    \\
    \kappa_1 &< \pi_n,\rho_1 \label{eq:APCP:pri:bra3}
    \\
    \rho_1 &< \pi_1+t_1 \label{eq:APCP:pri:sel1}
    \\
    \pi_n &< \pri_1+t_1,\pi_1+t_1 \label{eq:APCP:pri:bra2}
    \\
    \rho_n &< \pri_1+t_1,\pi_1+t_1 \label{eq:APCP:pri:bra1}
\end{align}
Each process $A_{i+1}$ for $0 \leq i < n$---thus including the leader---is typable as follows, assuming \cChan{$c_i$} is \cChan{$c_n$} for $i=0$:
\begin{align*}
    \emptyset &\vdash
    A_{i+1}
    \typInf \aChan{a_{i+1}}: \mu X . \oplus^{\pri_{i+1}} \{ \sff{start} : \&^{\kappa_{i+1}} \{ \sff{ack} : X \} \},
    \begin{array}[t]{@{}l@{}}
        \cChan{c_i}: \mu X . \&^{\pi_i} \{ \sff{start} : \&^{\rho_i} \{ \sff{next} : X \} \},
        \\
        \dChan{d_{i+1}}: \mu X . \oplus^{\pi_{i+1}} \{ \sff{start} : \oplus^{\rho_{i+1}} \{ \sff{next} : X \} \}
    \end{array}
\end{align*}
For each $1 \leq i \leq n$, the types for \cChan{$c_i$} and \dChan{$d_i$} are duals and are thus assigned equal priorities.

The priority requirements in the typing derivation of each $A_i$ are satisfiable.
The derivations of these processes have the following constraints:
\begin{itemize}

    \item
        For $A_1$ we require the inequalities listed above;

    \item
        For each $1 \leq i < n$, for $A_{i+1}$ we require $\pi_i < \pri_{i+1},\pi_{i+1}$, $\pri_{i+1} < \kappa_{i+1}$, $\pi_{i+1} < \rho_{i+1}$, $\kappa_{i+1} < \rho_i,\rho_{i+1}$, $\rho_i < \pri_{i+1}+t_{i+1},\rho_{i+1}$, and $\rho_{i+1} < \pi_{i+1}+t_{i+1}$.

\end{itemize}
We can satisfy these requirements by assigning $\pi_i := i$, $\pri_i := i+1$, $\kappa_i := i+2$, and $\rho_i := i+4$ for each $1 \leq i \leq n$, except with $\pi_n := n+3$ (to satisfy~\eqref{eq:APCP:pri:bra3} for $1 \leq n < 4$).
For the application of \ruleLabel{typ-rec}, each derivation also requires a common lifter $t_{i+1} > \max(\pri_{i+1}, \kappa_{i+1}, \pi_i, \rho_i, \pi_{i+1}, \rho_{i+1})$ for $0 \leq i < n$.
The priority requirements involving $t_{i+1}$ always require the priority lifted by $t_{i+1}$ to be higher than the priority not lifted by $t_{i+1}$, so the common lifter requirement easily satisfies these requirements.

Recall from \Cref{s:APCP:milner}:
\cyclicAsyncMilnerSched
Assuming given workers
\[
    \emptyset \vdash P_i \typInf \bChan{b_i}: \mu X . \&^{\pri_i} \{ \sff{start} : \oplus^{\kappa_i} \{ \sff{ack} : X \} \}
\]
for each $1 \leq i \leq n$, we have $\emptyset \vdash Sched_n \typInf \emptyset$.
Hence, it follows from \Cref{t:APCP:df} that $Sched_n$ is deadlock-free for each $n \geq 1$.

\subsection{Extensions: Explicit Session Closing and Replicated Servers}
\label{s:APCP:APCP:exts}

\begin{figure}[t]
    \begin{mathpar}
        \begin{bussproof}[typ-close]
            \bussAx{
                \Omega \vdash \pClose x[] \typInf x: \1^\pri
            }
        \end{bussproof}
        \and
        \begin{bussproof}[typ-wait]
            \bussAssume{
                \Omega \vdash P \typInf \Gamma
            }
            \bussAssume{
                \pri < \pr(\Gamma)
            }
            \bussBin{
                \Omega \vdash \pWait x() ; P \typInf \Gamma, x: \bot^\pri
            }
        \end{bussproof}
        \and
        \begin{bussproof}[typ-cli]
            \bussAssume{
                \pri < \pr(A)
            }
            \bussUn{
                \Omega \vdash \pCli x[y] \typInf x: {?}^\pri A, y: \ol{A}
            }
        \end{bussproof}
        \and
        \begin{bussproof}[typ-srv]
            \bussAssume{
                \Omega \vdash P \typInf {?}\Gamma, y: A
            }
            \bussAssume{
                \pri < \pr({?}\Gamma)
            }
            \bussBin{
                \Omega \vdash \pSrv x(y) ; P \typInf {?}\Gamma, x: {!}^\pri A
            }
        \end{bussproof}
        \and
        \begin{bussproof}[typ-cli$\ast$]
            \bussAssume{
                \Omega \vdash P \typInf \Gamma, y: A
            }
            \bussAssume{
                \pri < \pr(A)
            }
            \bussBin{
                \Omega \vdash \pCli* x[y] \cdot P \typInf \Gamma, x: {?}^\pri A
            }
        \end{bussproof}
        \and
        \begin{bussproof}[typ-weaken]
            \bussAssume{
                \Omega \vdash P \typInf \Gamma
            }
            \bussUn{
                \Omega \vdash P \typInf \Gamma, x: {?}^\pri A
            }
        \end{bussproof}
        \and
        \begin{bussproof}[typ-contract]
            \bussAssume{
                \Omega \vdash P \typInf \Gamma, x: {?}^\pri A, x': {?}^\kappa A
            }
            \bussAssume{
                \pi = \min(\pri,\kappa)
            }
            \bussBin{
                \Omega \vdash P \{ x/x' \} \typInf \Gamma, x: {?}^\pi A
            }
        \end{bussproof}
        \\ \dashes \\
        \begin{bussproof}[red-close-wait]
            \bussAx{
                \pRes{xy} ( \pClose x[] \| \pWait y() ; P) \redd P
            }
        \end{bussproof}
        \and
        \begin{bussproof}[red-cli-srv]
            \bussAx{
                \pRes{xy} ( \pCli x[a] \| \pSrv y(v) ; P \| Q ) \redd P \{ a/v \} \| \pRes{xy} ( \pSrv y(v) ; P \| Q )
            }
        \end{bussproof}
    \end{mathpar}

    \caption{Typing rules (top) and reductions (bottom) for explicit closing and replicated servers.}
    \label{f:APCP:exrules}
\end{figure}

As already mentioned, our presentation of \APCP does not include explicit closing and replicated servers.
Here, we briefly discuss what \APCP would look like with these constructs.

Explicit closing is useful in programming to be sure that all resources are cleaned up correctly.
There are several ways of integrating explicit closing in a calculus like \APCP.
Following, e.g., \cite{conf/concur/CairesP10,conf/icfp/Wadler12}, here we achieve explicit closing by adding closes (empty sends) $\pClose x[]$ and waits (empty receives) $\pWait x() ; P$ to the syntax in \ifthesis\Cref{d:APCP:APCP:syntax}\else\Cref{f:procdef} (top)\fi.
We also add the Rule~\ruleLabel{red-close-wait} to \ifthesis\Cref{d:APCP:APCP:redd}\else\Cref{f:procdef} (bottom)\fi.
At the level of types, we replace the conflated type $\bullet$ with $\1^\pri$ and $\bot^\pri$, associated to closes and waits, respectively.
Note that we do need priority annotations on types for closed names now, because wait is blocking and thus requires priority checks.
In the type system in \Cref{f:APCP:typingRules} (top), we replace Rule~\ruleLabel{typ-end} with Rules~\ruleLabel{type-close} and~\ruleLabel{typ-wait} in \Cref{f:APCP:exrules} (top).

For replicated servers, we add (asynchronous) client requests $\pCli x[y]$ and servers $\pSrv x(y) ; P$, with types ${?}^\pri A$ and ${!}^\pri A$, respectively.
We include syntactic sugar for binding client requests to continuations as in \Cref{n:APCP:sugar}: $\pCli* x[y] \cdot P := \pRes{ya} ( \pCli x[a] \| P )$.
Rule~\ruleLabel{red-cli-srv} is in \Cref{f:APCP:exrules} (bottom): it connects a client and a server and spawns a copy of the server.
Also, we add a structural congruence rule to clean up unused servers: $\pRes{xz} ( \pSrv x(y) ; P ) \equiv \0$.
In the type system, we add Rules~\ruleLabel{typ-cli}, \ruleLabel{typ-srv}, \ruleLabel{typ-weaken}, and~\ruleLabel{typ-contract} in \Cref{f:APCP:exrules} (top); the former two are for typing client requests and servers, respectively, and the latter two are for connecting to a server without requests and for multiple requests, respectively.
In Rule~\ruleLabel{typ-srv}, ${?}\Gamma$ means that every type in $\Gamma$ is of the form ${?}^\pri A$.
\Cref{f:APCP:exrules} (top) also includes a derivable Rule~\ruleLabel{typ-cli$\ast$} which types the syntactic sugar for bound client requests.

\section{An Intermezzo: From \APCP to \texorpdfstring{\protect\LAST}{LAST}}
\label{s:LASTn:LAST}

As attested by its expressivity and meta-theoretical results (\ifappendix\Cref{t:LASTn:APCP:tp,t:LASTn:APCP:df}\else{}type preservation and deadlock-freedom in \Cref{s:APCP:APCP:results} on \Cpageref{s:APCP:APCP:results}\fi), \APCP provides a convenient framework for analyzing asynchronous message-passing between cyclically connected processes.
In particular, \APCP provides a firm basis for designing languages with session-typed concurrency, asynchronous communication, and cyclic structures.
Ideally, we would like to faithfully compile any such language into \APCP, in order to transfer its correctness guarantees.

As discussed in the introduction, we look for answers in the realm of functional programming, in the form of variants of the \lamcalc with session-typed, message-passing concurrency.
In this context, Gay and Vasconcelos's \LAST~\cite{journal/jfp/GayV10} appears nicely positioned: \LAST is a call-by-value language in which programs consist of threads that are cyclically connected on channels that provide asynchronous message-passing (through buffers).
The calculus \LAST has spawned several variants that connect message-passing processes with message-passing functions, as we set out to do here.
Most notably, Walder~\cite{conf/icfp/Wadler12,journal/jfp/Wadler14} introduced \GV, a variant of \LAST with synchronous communication and non-cyclic thread connections, and gave a translation into his \CP (the synchronous ancestor of \APCP, without cyclic connections).
Subsequently, in the same spirit, Kokke and Dardha~\cite{conf/forte/KokkeD21} presented \PGV, a variant of \GV with cyclic  connections, and a translation from Dardha and Gay's \PCP (the synchronous ancestor of \APCP~\cite{conf/fossacs/DardhaG18}).

Hence, \LAST is a natural choice for a core programming language that can be studied via a translation into \APCP.
Notice that both \GV and \PGV enjoy deadlock-freedom by typing, whereas \LAST does not.
This strengthens our motivation for designing a translation into \APCP, as transferring the deadlock-freedom property from \APCP to \LAST would address a significant gap.
To our knowledge, such a \emph{translational} approach to ensuring deadlock-freedom for \LAST programs has not been achieved before.

In this section we gently recall \LAST and gradually introduce the key ingredients for its translation into \APCP.
For presentational purposes, we find it useful to present a variant of  \LAST that is more convenient towards a translation, denoted \LAST* (\Cref{s:LASTn:LAST:syntaxSemantics}).
The type system for \LAST*, described in \Cref{s:LASTn:LAST:typeSystem}, closely follows the one for \LAST in~\cite{journal/jfp/GayV10}.
Then, \ifappendix after briefly recalling \APCP and its meta-theoretical properties in \Cref{s:LASTn:APCP}, \fi we discuss in \Cref{s:LASTn:LAST:trans} the potential design of a translation from \LAST* into \APCP.
We purposefully use the word ``potential'': we seek a translation that is \emph{faithful}, i.e., that preserves and reflects behaviors in a precise sense.
Given this focus, we shall argue that the call-by-value semantics of \LAST* is not well-suited for inducing a faithful translation.
As such, this section serves as motivation for introducing \LASTn, a  variant of \LAST based on a call-by-name semantics that enjoys a faithful translation into \APCP and admits the translational approach to the deadlock-freedom property (\Cref{s:LASTn:LASTn}).

\subsection{The Syntax and Semantics of \texorpdfstring{\protect\LAST*}{LAST*}}
\label{s:LASTn:LAST:syntaxSemantics}

In \LAST*, programs consist of two layers: while \emph{terms} ($\term{M,N,\ldots}$) define the behavior of threads, \emph{configurations} ($\term{C,D,\ldots}$) are obtained by composing threads in parallel, so as to enable their interaction by  exchanging messages on buffered channels.
Before detailing these syntactic elements (and their semantics), we point out the differences between \LAST and \LAST*.
While in \LAST endpoints are created by synchronization on shared names, \LAST* features a dedicated construct for endpoint creation (denoted~$\term{\tNew}$).
Also, thread creation in \LAST*  involves an explicit continuation, not present in \LAST.
Moreover,  in \LAST buffers run next to threads, whereas they are integrated in endpoint restrictions in \LAST*.
Finally, for simplicity, \LAST* accounts linear sessions only.


\paragraph{Terms.}

\begin{figure}[t!]
    Terms ($\term{M},\term{N},\term{L}$), values ($\term{v}$), and reduction contexts ($\tCtx{R}$):
    \begin{align*}
        \term{M}\mathrlap{,\term{N},\term{L}} &
        \\[-3mm]
        & \begin{array}[t]{@{}l@{\kern.5ex}lr@{\kern2ex}l@{\kern.5ex}lr@{}}
            {} ::= &
            \term{x} & \text{variable}
            & \sepr &
            \term{\tNew} & \text{create new channel}
            \\ \sepr* &
            \term{()} & \text{unit value}
            & \sepr &
            \term{\tSpawn M ; N} & \text{spawn $\term{M}$ in parallel to $\term{N}$}
            \\ \sepr* &
            \term{\lam x . M} & \text{abstraction}
            & \sepr &
            \term{(M,N)} & \text{pair construction}
            \\ \sepr* &
            \term{M~N} & \text{application}
            & \sepr &
            \term{\tLet (x,y)=M \tIn N} & \text{pair deconstruction}
            \\ \sepr* &
            \term{\tSend M \, N} & \text{send $\term{M}$ along $\term{N}$}
            & \sepr &
            \term{\tSel \ell \, M} & \text{select label $\ell$ along $\term{M}$}
            \\ \sepr* &
            \term{\tRecv M} & \text{receive along $\term{M}$}
            & \sepr &
            \term{\tCase M \tOf \{ i : N \}_{i \in I}} & \text{offer labels in $I$ along $\term{M}$}
        \end{array}
        \\
        \term{v} &::=
        \term{x}
        \sepr
        \term{\lam x . M}
        \sepr
        \term{(v,v)}
        \sepr
        \term{()}
        \\
        \term{\tCtx{R}}
        & \begin{array}[t]{@{}l@{\kern.5ex}l@{}}
            {} ::= &
            \term{\tHole}
            \sepr
            \term{\tCtx{R}~M}
            \sepr
            \term{v~\tCtx{R}}
            \sepr
            \term{(\tCtx{R},M)}
            \sepr
            \term{(M,\tCtx{R})}
            \sepr
            \term{\tLet (x,y) = \tCtx{R} \tIn M}
            \\ \sepr* &
            \term{\tSpawn \tCtx{R} ; M}
            \sepr
            \term{\tSend M \, \tCtx{R}}
            \sepr
            \term{\tRecv \tCtx{R}}
            \sepr
            \term{\tSel \ell\, \tCtx{R}}
            \sepr
            \term{\tCase \tCtx{R} \tOf \{ i : M \}_{i \in I}}
        \end{array}
    \end{align*}

    \dashes

    Term reduction ($\reddM$):
    \begin{mathpar}
        \begin{bussproof}[red-lam]
            \bussAx{
                \term{(\lam x . M)~v}
                \reddM
                \term{M \{ v/x \}}
            }
        \end{bussproof}
        \and
        \begin{bussproof}[red-pair]
            \bussAx{
                \term{\tLet (x,y) = (v_1,v_2) \tIn M}
                \reddM
                \term{M \{ v_1/x,v_2/y \}}
            }
        \end{bussproof}
        \and
        \begin{bussproof}[red-lift]
            \bussAssume{
                \term{M}
                \reddM
                \term{N}
            }
            \bussUn{
                \tCtx{R}[M]
                \reddM
                \tCtx{R}[N]
            }
        \end{bussproof}
    \end{mathpar}

    \caption{The \protect\LAST* term language.}\label{f:LASTn:LAST:terms}
\end{figure}

\Cref{f:LASTn:LAST:terms} (top) gives the term syntax for \LAST*.
The functional behavior of terms is defined by standard \lamcalc constructs for variables $\term{x}$, the unit value $\term{()}$, abstraction $\term{\lam x . M}$ and application $\term{M\ N}$, and pair construction $\term{(M,N)}$ and deconstruction $\term{\tLet (x,y) = M \tIn N}$.
As usual, to improve readability, we often write $\tLet x = M \tIn N$ to denote $(\lam x . N)~M$.

The remaining constructs define the thread and message-passing behavior of terms; their exact semantics will be defined for configurations, so we briefly describe them here.
Construct $\term{\tNew}$ creates a new buffered channel with two endpoints (referred to with variables).
Construct $\term{\tSpawn M ; N}$ spawns a thread running~$\term{M}$ and continues as~$\term{N}$.
Constructs $\term{\tSend N M}$ and $\term{\tRecv M}$ denote sending and receiving messages along $\term{M}$ once it has reduced to a variable referring to a buffer endpoint, respectively; that is, sending entails placing the message $\term{N}$ at the end of the buffer, and receiving entails taking a message from the start of the buffer (if there).
Constructs $\term{\tSel \ell \, M}$ and $\term{\tCase M \tOf \{ i : N_i \}_{i \in I}}$ denote selecting and offering labels along $\term{M}$ once it has reduced to an endpoint variable, respectively; that is, selection entails placing the label $\term{\ell}$ at the end of the buffer, and offering entails taking a label~\mbox{$\term{j} \in I$} from the start of the buffer (if there) and continuing in the corresponding branch~$\term{N_j}$.

Following~\cite{journal/jfp/GayV10}, the term reduction semantics of \LAST* employs a call-by-value (CbV) approach; \Cref{f:LASTn:LAST:terms} (center) defines values $\term{v}$, i.e., terms that cannot further reduce on their own: variables, abstractions, pairs, and the unit value.
\Cref{f:LASTn:LAST:terms} (center) also defines reduction contexts $\tCtx{R}$ that define under which positions subterms may reduce.
Finally, \Cref{f:LASTn:LAST:terms} (bottom) defines term reduction ($\reddM$).
Rule~\ruleLabel{red-lam} reduces an abstraction applied to a value; the substitution of a value $\term{v}$ for a (free) variable $\term{x}$ is denoted~$\term{\{ v/x \}}$, as usual.
Rule~\ruleLabel{red-pair} reduces the pair deconstruction of a pair of values to two substitutions.
Rule~\ruleLabel{red-lift} closes term reduction under reduction contexts.

\begin{example}
    \label{x:LASTn:CbV}
    The following term behavior illustrate the CbV semantics of \LAST*:
    \[
        \term{\big(\lam x . x~(\lam y . y)\big)~\big((\lam w . w)~(\lam z . z)\big)}
        \reddM
        \term{\big(\lam x . x~(\lam y . y)\big)~(\lam z . z)}
        \reddM
        \term{(\lam z . z)~(\lam y . y)}
        \reddM
        \term{\lam y . y}
        \tag*{\qedhere}
    \]
\end{example}

\noindent
We will illustrate the message-passing behavior of \LAST* after introducing configurations.

\paragraph{Configurations.}

\begin{figure}[p]
    \def\MathparLineskip{\lineskip=3pt}
    Markers ($\term{\phi}$), messages ($\term{m},\term{n}$), configurations ($\term{C},\term{D},\term{E}$), thread contexts ($\tCtx{F}$), configuration  contexts ($\tCtx{G}$):
    \begin{align*}
        \term{\phi} ::= {}
        & \term{\tMain} \sepr \term{\tChild}
        &
        \term{m},\term{n} ::= {}
        & \term{v} \sepr \term{\ell}
        \\
        \term{C},\term{D},\term{E} ::= {}
        & \term{\phi\, M}  \sepr \term{\pRes{x\tBfr{\vec{m}}y} C} \sepr \term{C \prl D}
        \\
        \tCtx{F} ::= {}
        & \term{\phi\, \tCtx{R}}
        &
        \term{\tCtx{G}} ::= {}
        & \term{\tHole} \sepr \term{\tCtx{G} \prl C} \sepr \term{\pRes{x\tBfr{\vec{m}}y} \tCtx{G}}
    \end{align*}

    \dashes

    Structural congruence for configurations ($\equivC$):

    \vspace{-5ex}
    \begin{mathpar}
        \begin{bussproof}[sc-res-swap]
            \bussAx{
                \term{\pRes{x\tBfr{\epsi}y} C}
                \equivC
                \term{\pRes{y\tBfr{\epsi}x} C}
            }
        \end{bussproof}
        \and
        \begin{bussproof}[sc-res-comm]
            \bussAx{
                \term{\pRes{x\tBfr{\vec{m}}y} \pRes{z\tBfr{\vec{n}}w} C}
                \equivC
                \term{\pRes{z\tBfr{\vec{n}}w} \pRes{x\tBfr{\vec{m}}y} C}
            }
        \end{bussproof}
        \and
        \begin{bussproof}[sc-res-ext]
            \bussAssume{
                \term{x},\term{y} \notin \fv(\term{C})
            }
            \bussUn{
                \term{\pRes{x\tBfr{\vec{m}}y} ( C \prl D )}
                \equivC
                \term{C \prl \pRes{x\tBfr{\vec{m}}y} D}
            }
        \end{bussproof}
        \and
        \begin{bussproof}[sc-par-comm]
            \bussAx{
                \term{C \prl D}
                \equivC
                \term{D \prl C}
            }
        \end{bussproof}
        \and
        \begin{bussproof}[sc-par-assoc]
            \bussAx{
                \term{C \prl (D \prl E)}
                \equivC
                \term{(C \prl D) \prl E}
            }
        \end{bussproof}
    \end{mathpar}

    \dashes

    Configuration reduction ($\reddC$):

    \vspace{-5ex}
    \begin{mathpar}
        \begin{bussproof}[red-new]
            \bussAx{
                \tCtx{F}[\tNew]
                \reddC
                \term{\pRes{x\tBfr{\epsi}y} ( \tCtx{F}[(x,y)] )}
            }
        \end{bussproof}
        \and
        \begin{bussproof}[red-spawn]
            \bussAx{
                \tCtx{F}[\tSpawn M ; N]
                \reddC
                \term{\tCtx{F}[N] \prl \tChild\,M}
            }
        \end{bussproof}
        \and
        \begin{bussproof}[red-send]
            \bussAx{
                \term{\pRes{x\tBfr{\vec{m}}y} ( \tCtx{F}[\tSend v \, x] \prl C )}
                \reddC
                \term{\pRes{x\tBfr{v,\vec{m}}y} ( \tCtx{F}[x] \prl C )}
            }
        \end{bussproof}
        \and
        \begin{bussproof}[red-recv]
            \bussAx{
                \term{\pRes{x\tBfr{\vec{m},v}y} ( \tCtx{F}[\tRecv y] \prl C )}
                \reddC
                \term{\pRes{x\tBfr{\vec{m}}y} ( \tCtx{F}[(v,y)] \prl C )}
            }
        \end{bussproof}
        \and
        \begin{bussproof}[red-select]
            \bussAx{
                \term{\pRes{x\tBfr{\vec{m}}y} ( \tCtx{F}[\tSel \ell\, x] \prl C )}
                \reddC
                \term{\pRes{x\tBfr{\ell,\vec{m}}y} ( \tCtx{F}[x] \prl C )}
            }
        \end{bussproof}
        \and
        \begin{bussproof}[red-case]
            \bussAssume{
                j \in I
            }
            \bussUn{
                \term{\pRes{x\tBfr{\vec{m},j}y} ( \tCtx{F}[\tCase y \tOf \{ i : M_i \}_{i \in I}] \prl C )}
                \reddC
                \term{\pRes{x\tBfr{\vec{m}}y} ( \tCtx{F}[M_j~y] \prl C )}
            }
        \end{bussproof}
        \and
        \begin{bussproof}[red-res-nil]
            \bussAssume{
                \term{x},\term{y} \notin \fv(\term{C})
            }
            \bussUn{
                \term{\pRes{x\tBfr{\epsi}y} C} \reddC \term{C}
            }
        \end{bussproof}
        \and
        \begin{bussproof}[red-par-nil]
            \bussAx{
                \term{C \prl \tChild \, ()} \reddC \term{C}
            }
        \end{bussproof}
        \and
        \begin{bussproof}[red-lift-C]
            \bussAssume{
                \term{C}
                \reddC
                \term{C'}
            }
            \bussUn{
                \tCtx{G}[C]
                \reddC
                \tCtx{G}[C']
            }
        \end{bussproof}
        \and
        \begin{bussproof}[red-lift-M]
            \bussAssume{
                \term{M}
                \reddM
                \term{M'}
            }
            \bussUn{
                \tCtx{F}[M]
                \reddC
                \tCtx{F}[M']
            }
        \end{bussproof}
        \and
        \begin{bussproof}[red-conf-lift-sc]
            \bussAssume{
                \term{C}
                \equivC
                \term{C'}
            }
            \bussAssume{
                \term{C'}
                \reddC
                \term{D'}
            }
            \bussAssume{
                \term{D'}
                \equivC
                \term{D}
            }
            \bussTern{
                \term{C}
                \reddC
                \term{D}
            }
        \end{bussproof}
    \end{mathpar}
    \caption{The \protect\LAST* configuration language: syntax and semantics.}\label{f:LASTn:LAST:confs}
\end{figure}

Functional calculi such as \LAST* have a clear distinction between the static and dynamic parts of their languages: a \LAST* program starts as a closed functions (the static part) and evolves into several threads operating in parallel and communicating through message-passing (the dynamic part).
We refer to the static and dynamic parts of \LAST* as terms and configurations, respectively.
In general, one writes a \LAST* program as a single main thread containing a term that spawns and connects child threads.
To contrast, process calculi such as \APCP blur the lines between such static and dynamic parts, as \APCP ``programs'' are immediately written as configurations of parallel subprocesses connected on channels.

\Cref{f:LASTn:LAST:confs} gives the configuration syntax for \LAST*.
Given a term $\term{M}$, the configuration~$\term{\phi \, M}$ denotes a corresponding thread, where the marker $\term{\phi}$ is useful to distinguish the main thread $\term{\tMain}$ from child threads $\term{\tChild}$---this distinction will be useful for typing.
Buffered channels are denoted $\term{\pRes{x\tBfr{\vec{m}}y} C}$.
Here, $\term{C}$ has access to the endpoints $\term{x}$ and $\term{y}$.
The buffer itself $\term{x\tBfr{\vec{m}}y}$ is an ordered sequence of messages (values and labels, denoted $\term{\vec{m}}$) sent on~$\term{x}$ and to be received on $\term{y}$.
This means that $\term{C}$ may send/select on $\term{x}$ and receive/offer on~$\term{y}$.
Once the buffer is empty (i.e., $\term{\vec{m}} = \term{\epsi}$), $\term{x}$ and $\term{y}$ may switch roles.
Configuration~$\term{C \prl D}$ denotes the parallel composition of~$\term{C}$ and $\term{D}$.
\Cref{f:LASTn:LAST:confs} (top) also defines thread contexts~$\tCtx{F}$ as term reduction contexts inside threads, and configuration contexts $\tCtx{G}$.

The reduction semantics for configurations is defined on specific arrangements of buffers and threads.
To ensure such arrangements, we define a structural congruence for configurations ($\equivC$), the least congruence on configurations satisfying the rules in \Cref{f:LASTn:LAST:confs} (center).
Rule~\ruleLabel{sc-res-swap} swaps the direction of buffered channels and thus the input/output roles of the channel's endpoints.
Rule~\ruleLabel{sc-res-comm} defines commutativity of buffered channels.
Rule~\ruleLabel{sc-res-ext} defines scope extrusion/inclusion for buffered channels.
Rules~\ruleLabel{sc-par-comm} and~\ruleLabel{sc-par-assoc} define commutativity and associativity for parallel composition.

\Cref{f:LASTn:LAST:confs} (bottom) gives the reduction semantics for configurations ($\reddC$).
It defines how terms in threads interact with each other by exchanging messages through buffered channels.
Rule~\ruleLabel{red-new} reduces a $\term{\tNew}$ construct in a thread by creating a new buffered channel and returning the endpoints $\term{x}$ and $\term{y}$.
Rule~\ruleLabel{red-spawn} reduces a~$\term{\tSpawn}$ construct in a thread by creating a new child thread.
Rule~\ruleLabel{red-send} reduces a~$\term{\tSend}$ by placing the value at the end of the enclosing buffer and returning the endpoint.
Dually, Rule~\ruleLabel{red-recv} reduces a $\term{\tRecv}$ by retrieving the value at the start of the enclosing buffer and returning it along with the endpoint.
Rule~\ruleLabel{red-select} reduces a $\term{\tSel}$ by placing the label at the end of the enclosing buffer and returning the endpoint.
Dually, Rule~\ruleLabel{red-case} reduces a $\term{\tCase}$ by retrieving the label at the start of the enclosing buffer and applying the label's corresponding continuation to the endpoint.
Rule~\ruleLabel{red-res-nil} garbage collects buffers that are no longer used, and Rule~\ruleLabel{red-par-nil} garbage collects child threads that have reduced to unit.
Rules~\ruleLabel{red-lift-C} and~\ruleLabel{red-lift-M} close configuration reduction under configuration contexts and enable terms in threads to reduce, respectively.
Rule~\ruleLabel{red-conf-lift-SC} closes configuration reduction under structural congruence.

\begin{example}[The Bookshop Scenario, Revisited]
    \label{x:LASTn:LAST:Shop}
    We illustrate the message-passing behavior of \LAST* by considering the bookshop example from \Cref{s:LASTn:example}.
    Note that \LAST* does not have $\term{\tClose}$-constructs: we will motivate them in \Cref{s:LASTn:LASTn}; here, for simplicity, we consider the system in \Cref{s:LASTn:example} without these constructs.

    First, let us explore how the system sets up some channels and threads:
    \begin{align*}
        \term{\tMain \, \sff{Sys}}
        &=
        \term{\tMain \, \tLet (s,s') = \tNew \tIn \tSpawn \sff{Shop}_s ; \tLet (m,m') = \tNew \tIn \tSpawn \sff{Mother}_m ; \sff{Son}_{s',m'}}
        \\
        &\reddC
        \term{\pRes{y\tBfr{\epsi}y'} \big( \tMain \, \begin{array}[t]{@{}l@{}}
            \tLet (s,s') = (y,y') \tIn \tSpawn \sff{Shop}_s ;
            \\
            \tLet (m,m') = \tNew \tIn \tSpawn \sff{Mother}_m ; \sff{Son}_{s',m'} \big)
        \end{array}}
        \\
        &\reddC
        \term{\pRes{y\tBfr{\epsi}y'} \big( \tMain \, \tSpawn \sff{Shop}_y ; \tLet (m,m') = \tNew \tIn \tSpawn \sff{Mother}_m ; \sff{Son}_{y',m'} \big)}
        \\
        &\reddC
        \term{\pRes{y\tBfr{\epsi}y'} \big( \tMain \, \tLet (m,m') = \tNew \tIn \tSpawn \sff{Mother}_m ; \sff{Son}_{y',m'} \prl \tChild \, \sff{Shop}_y \big)}
        \\
        &\reddC
        \term{\pRes{y\tBfr{\epsi}y'} \big( \pRes{z\tBfr{\epsi}z'} \big( \tMain \, \tLet (m,m') = (z,z') \tIn \tSpawn \sff{Mother}_m ; \sff{Son}_{y',m'} \big) \prl \tChild \, \sff{Shop}_y \big)}
        \\
        &\reddC
        \term{\pRes{y\tBfr{\epsi}y'} \big( \pRes{z\tBfr{\epsi}z'} \big( \tMain \, \tSpawn \sff{Mother}_z ; \sff{Son}_{y',z'} \big) \prl \tChild \, \sff{Shop}_y \big)}
        \\
        &\reddC
        \term{\pRes{y\tBfr{\epsi}y'} \big( \pRes{z\tBfr{\epsi}z'} \big( \tMain \, \sff{Son}_{y',z'} \prl \tChild \, \sff{Mother}_z \big) \prl \tChild \, \sff{Shop}_y \big)}
    \end{align*}

    \noindent
    Next, let us see how the son sends the book title and his choice to buy to the shop, and then his connection with the shop to his mother, without waiting for his messages to be received:
    \begin{align*}
        &\equivC
        \term{\pRes{y'\tBfr{\epsi}y} \big( \pRes{z\tBfr{\epsi}z'} \big( \tMain \, \tLet s'_1 = \tSend \text{``Dune''} \, y' \tIn \ldots \prl \tChild \, \sff{Mother}_z \big) \prl \tChild \, \sff{Shop}_y \big)}
        \\ \displaybreak[1]
        &\reddC
        \term{\pRes{y'\tBfr{\text{``Dune''}}y} \big( \pRes{z\tBfr{\epsi}z'} \big( \tMain \, \tLet s'_1 = y' \tIn \ldots \prl \tChild \, \sff{Mother}_z \big) \prl \tChild \, \sff{Shop}_y \big)}
        \\ \displaybreak[1]
        &\reddC
        \term{\pRes{y'\tBfr{\text{``Dune''}}y} \big( \pRes{z\tBfr{\epsi}z'} \big( \tMain \, \tLet s'_2 = \tSel \sff{buy} \, y' \tIn \ldots \prl \tChild \, \sff{Mother}_z \big) \prl \tChild \, \sff{Shop}_y \big)}
        \\ \displaybreak[1]
        &\reddC
        \term{\pRes{y'\tBfr{\sff{buy},\text{``Dune''}}y} \big( \pRes{z\tBfr{\epsi}z'} \big( \tMain \, \tLet s'_2 = y' \tIn \ldots \prl \tChild \, \sff{Mother}_z \big) \prl \tChild \, \sff{Shop}_y \big)}
        \\ \displaybreak[1]
        &\reddC
        \term{\pRes{y'\tBfr{\sff{buy},\text{``Dune''}}y} \big( \pRes{z\tBfr{\epsi}z'} \big( \tMain \, \tLet m'_1 = \tSend y' \, z' \tIn \ldots \prl \tChild \, \sff{Mother}_z \big) \prl \tChild \, \sff{Shop}_y \big)}
        \\ \displaybreak[1]
        &\reddC
        \term{\pRes{y'\tBfr{\sff{buy},\text{``Dune''}}y} \big( \pRes{z'\tBfr{y'}z} \big( \tMain \, \tLet m'_1 = z' \tIn \ldots \prl \tChild \, \sff{Mother}_z \big) \prl \tChild \, \sff{Shop}_y \big)}
        \\ \displaybreak[1]
        &\reddC
        \term{\pRes{y'\tBfr{\sff{buy},\text{``Dune''}}y} \big( \begin{array}[t]{@{}l@{}}
            \pRes{z'\tBfr{y'}z} \big( \tMain \, \tLet (\textit{book},m'_2) = \tRecv z' \tIn \textit{book} \prl \tChild \, \sff{Mother}_z \big) 
            \\
            \prl \tChild \, \sff{Shop}_y \big)
        \end{array}}
    \end{align*}

    \noindent
    Now, we see the mother receive the shop's connection and send her credit card information:
    \begin{align*}
        &=
        \term{
            \pRes{y'\tBfr{\sff{buy},\text{``Dune''}}y} \big(
                \pRes{z'\tBfr{y'}z} \big(
                    \begin{array}[t]{@{}l@{}}
                        \tMain \, \tLet (\textit{book},m'_2) = \tRecv z' \tIn \textit{book}
                        \\
                        \prl \tChild \, \tLet (x,m_1) = \tRecv z \tIn \ldots
                    \big)
                    \prl \tChild \, \sff{Shop}_y
                \big)
            \end{array}
        }
        \\
        &\reddC
        \term{
            \pRes{y'\tBfr{\sff{buy},\text{``Dune''}}y} \big(
                \pRes{z\tBfr{\epsi}z'} \big(
                    \begin{array}[t]{@{}l@{}}
                        \tMain \, \tLet (\textit{book},m'_2) = \tRecv z' \tIn \textit{book}
                        \\
                        \prl \tChild \, \tLet (x,m_1) = (y',z) \tIn \ldots
                    \big)
                    \prl \tChild \, \sff{Shop}_y
                \big)
            \end{array}
        }
        \\
        &\reddC
        \term{
            \pRes{y'\tBfr{\sff{buy},\text{``Dune''}}y} \big(
                \pRes{z\tBfr{\epsi}z'} \big(
                    \begin{array}[t]{@{}l@{}}
                        \tMain \, \tLet (\textit{book},m'_2) = \tRecv z' \tIn \textit{book}
                        \\
                        \prl \tChild \, \tLet x_1 = \tSend \text{visa} \, y' \tIn \ldots
                    \big)
                    \prl \tChild \, \sff{Shop}_y
                \big)
            \end{array}
        }
        \\
        &\reddC
        \term{
            \pRes{y'\tBfr{\text{visa},\sff{buy},\text{``Dune''}}y} \big(
                \pRes{z\tBfr{\epsi}z'} \big(
                    \begin{array}[t]{@{}l@{}}
                        \tMain \, \tLet (\textit{book},m'_2) = \tRecv z' \tIn \textit{book}
                        \\
                        \prl \tChild \, \tLet x_1 = y' \tIn \ldots
                    \big)
                    \prl \tChild \, \sff{Shop}_y
                \big)
            \end{array}
        }
        \\
        &\reddC
        \term{
            \pRes{y'\tBfr{\text{visa},\sff{buy},\text{``Dune''}}y} \big(
                \pRes{z\tBfr{\epsi}z'} \big(
                    \begin{array}[t]{@{}l@{}}
                        \tMain \, \tLet (\textit{book},m'_2) = \tRecv z' \tIn \textit{book}
                        \\
                        \prl \tChild \, \tLet (\textit{book},x_2) = \tRecv y' \tIn \ldots
                    \big)
                    \prl \tChild \, \sff{Shop}_y
                \big)
            \end{array}
        }
    \end{align*}

    \noindent
    Finally, the sequence of reductions in \Cref{f:LASTn:x:LAST:Shop} shows how the shop reads  messages and sends the book, and how the mother forwards it to her son.
\end{example}

\begin{figure}[p]
    \begin{align*}
        &=
        \term{
            \pRes{y'\tBfr{\text{visa},\sff{buy},\text{``Dune''}}y} \big(
                \begin{array}[t]{@{}l@{}}
                    \pRes{z\tBfr{\epsi}z'} \big(
                        \begin{array}[t]{@{}l@{}}
                            \tMain \, \tLet (\textit{book},m'_2) = \tRecv z' \tIn \textit{book}
                            \\
                            \prl \tChild \, \tLet (\textit{book},x_2) = \tRecv y' \tIn \ldots
                        \big)
                    \end{array}
                    \\
                    \prl \tChild \, \tLet (\textit{title},s_1) = \tRecv y \tIn \ldots
                \big)
            \end{array}
        }
        \\
        &\reddC
        \term{
            \pRes{y'\tBfr{\text{visa},\sff{buy}}y} \big(
                \begin{array}[t]{@{}l@{}}
                    \pRes{z\tBfr{\epsi}z'} \big( \begin{array}[t]{@{}l@{}}
                        \tMain \, \tLet (\textit{book},m'_2) = \tRecv z' \tIn \textit{book}
                        \\
                        \prl \tChild \, \tLet (\textit{book},x_2) = \tRecv y' \tIn \ldots
                    \big) \end{array}
                    \\
                    \prl \tChild \, \tLet (\textit{title},s_1) = (\text{``Dune''},y) \tIn \ldots
                \big)
            \end{array}
        }
        \\
        &\reddC
        \term{
            \pRes{y'\tBfr{\text{visa},\sff{buy}}y} \big(
                    \pRes{z\tBfr{\epsi}z'} \big( \begin{array}[t]{@{}l@{}}
                        \tMain \, \tLet (\textit{book},m'_2) = \tRecv z' \tIn \textit{book}
                        \\
                        \prl \tChild \, \tLet (\textit{book},x_2) = \tRecv y' \tIn \ldots
                    \big)
                    \prl \tChild \, \tCase y \tOf \{\ldots\}
                \big)
            \end{array}
        }
        \\
        &\reddC
        \term{
            \pRes{y'\tBfr{\text{visa}}y} \big(
                    \pRes{z\tBfr{\epsi}z'} \big( \begin{array}[t]{@{}l@{}}
                        \tMain \, \tLet (\textit{book},m'_2) = \tRecv z' \tIn \textit{book}
                        \\
                        \prl \tChild \, \tLet (\textit{book},x_2) = \tRecv y' \tIn \ldots
                    \big)
                    \prl \tChild \, ( \lam s_2 \ldots )~y
                \big)
            \end{array}
        }
        \\
        &\reddC
        \term{
            \pRes{y'\tBfr{\text{visa}}y} \big(
                \begin{array}[t]{@{}l@{}}
                    \pRes{z\tBfr{\epsi}z'} \big( \begin{array}[t]{@{}l@{}}
                        \tMain \, \tLet (\textit{book},m'_2) = \tRecv z' \tIn \textit{book}
                        \\
                        \prl \tChild \, \tLet (\textit{book},x_2) = \tRecv y' \tIn \ldots
                    \big) \end{array}
                    \\
                    \prl \tChild \, \tLet (\textit{card},s_3) = \tRecv y \tIn \ldots
                \big)
            \end{array}
        }
        \\
        &\reddC
        \term{
            \pRes{y\tBfr{\epsi}y'} \big(
                \begin{array}[t]{@{}l@{}}
                    \pRes{z\tBfr{\epsi}z'} \big( \begin{array}[t]{@{}l@{}}
                        \tMain \, \tLet (\textit{book},m'_2) = \tRecv z' \tIn \textit{book}
                        \\
                        \prl \tChild \, \tLet (\textit{book},x_2) = \tRecv y' \tIn \ldots
                    \big) \end{array}
                    \\
                    \prl \tChild \, \tLet (\textit{card},s_3) = (\text{visa},y) \tIn \ldots
                \big)
            \end{array}
        }
        \\
        &\reddC
        \term{
            \pRes{y\tBfr{\epsi}y'} \big(
                \begin{array}[t]{@{}l@{}}
                    \pRes{z\tBfr{\epsi}z'} \big( \begin{array}[t]{@{}l@{}}
                        \tMain \, \tLet (\textit{book},m'_2) = \tRecv z' \tIn \textit{book}
                        \\
                        \prl \tChild \, \tLet (\textit{book},x_2) = \tRecv y' \tIn \ldots
                    \big) \end{array}
                    \\
                    \prl \tChild \, \tLet s_4 = \tSend \text{book}(\text{``Dune''}) \, y \tIn ()
                \big)
            \end{array}
        }
        \\
        &\reddC
        \term{
            \pRes{y\tBfr{\text{book}(\text{``Dune''})}y'} \big(
                \begin{array}[t]{@{}l@{}}
                    \pRes{z\tBfr{\epsi}z'} \big( \begin{array}[t]{@{}l@{}}
                        \tMain \, \tLet (\textit{book},m'_2) = \tRecv z' \tIn \textit{book}
                        \\
                        \prl \tChild \, \tLet (\textit{book},x_2) = \tRecv y' \tIn \ldots
                    \big) \end{array}
                    \\
                    \prl \tChild \, \tLet s_4 = y \tIn ()
                \big)
            \end{array}
        }
        \\
        &\reddC
        \term{
            \pRes{y\tBfr{\text{book}(\text{``Dune''})}y'} \big(
                    \pRes{z\tBfr{\epsi}z'} \big( \begin{array}[t]{@{}l@{}}
                        \tMain \, \tLet (\textit{book},m'_2) = \tRecv z' \tIn \textit{book}
                        \\
                        \prl \tChild \, \tLet (\textit{book},x_2) = \tRecv y' \tIn \ldots
                    \big)
                    \prl \tChild \, ()
                \big)
            \end{array}
        }
        \\
        &\reddC
        \term{
            \pRes{y\tBfr{\text{book}(\text{``Dune''})}y'} \big(
                \pRes{z\tBfr{\epsi}z'} \big( \begin{array}[t]{@{}l@{}}
                    \tMain \, \tLet (\textit{book},m'_2) = \tRecv z' \tIn \textit{book}
                    \\
                    \prl \tChild \, \tLet (\textit{book},x_2) = \tRecv y' \tIn \ldots
                \big)
            \big) \end{array}
        }
        \\
        &\reddC
        \term{
            \pRes{y\tBfr{\epsi}y'} \big(
                \pRes{z\tBfr{\epsi}z'} \big( \begin{array}[t]{@{}l@{}}
                    \tMain \, \tLet (\textit{book},m'_2) = \tRecv z' \tIn \textit{book}
                    \\
                    \prl \tChild \, \tLet (\textit{book},x_2) = (\text{book}(\text{``Dune''}),y') \tIn \ldots
                \big)
            \big) \end{array}
        }
        \\ \displaybreak[1]
        &\reddC
        \term{
            \pRes{y\tBfr{\epsi}y'} \big(
                \pRes{z\tBfr{\epsi}z'} \big( \begin{array}[t]{@{}l@{}}
                    \tMain \, \tLet (\textit{book},m'_2) = \tRecv z' \tIn \textit{book}
                    \\
                    \prl \tChild \, \tLet m_2 = \tSend \text{book}(\text{``Dune''}) \, z \tIn ()
                \big)
            \big) \end{array}
        }
        \\ \displaybreak[1]
        &\reddC
        \term{
            \pRes{z\tBfr{\epsi}z'} \big(
                \tMain \, \tLet (\textit{book},m'_2) = \tRecv z' \tIn \textit{book}
                \prl \tChild \, \tLet m_2 = \tSend \text{book}(\text{``Dune''}) \, z \tIn ()
            \big)
        }
        \\ \displaybreak[1]
        &\reddC
        \term{
            \pRes{z\tBfr{\text{book}(\text{``Dune''})}z'} \big(
                \tMain \, \tLet (\textit{book},m'_2) = \tRecv z' \tIn \textit{book}
                \prl \tChild \, \tLet m_2 = z \tIn ()
            \big)
        }
        \\ \displaybreak[1]
        &\reddC
        \term{
            \pRes{z\tBfr{\text{book}(\text{``Dune''})}z'} \big(
                \tMain \, \tLet (\textit{book},m'_2) = \tRecv z' \tIn \textit{book}
                \prl \tChild \, ()
            \big)
        }
        \\ \displaybreak[1]
        &\reddC
        \term{
            \pRes{z\tBfr{\text{book}(\text{``Dune''})}z'} \big(
                \tMain \, \tLet (\textit{book},m'_2) = \tRecv z' \tIn \textit{book}
            \big)
        }
        \\ \displaybreak[1]
        &\reddC
        \term{
            \pRes{z\tBfr{\epsi}z'} \big(
                \tMain \, \tLet (\textit{book},m'_2) = (\text{book}(\text{``Dune''}),z') \tIn \textit{book}
            \big)
        }
        \\ \displaybreak[1]
        &\reddC
        \term{
            \pRes{z\tBfr{\epsi}z'} \big(
                \tMain \, \text{book}(\text{``Dune''})
            \big)
        }
        \reddC
        \term{
            \tMain \, \text{book}(\text{``Dune''})
        }
        \tag*{\qedhere}
    \end{align*}

    \vspace{-9pt}

    \caption{Reduction sequence from \Cref{x:LASTn:LAST:Shop} \label{f:LASTn:x:LAST:Shop}}
\end{figure}


\subsection{The Type System of \texorpdfstring{\protect\LAST*}{LAST*}}
\label{s:LASTn:LAST:typeSystem}

\LAST* includes functional types for functions and pairs and session types for communication.
The syntax and meaning of functional types ($\type{T},\type{U}$) and session types ($\type{S}$) are as follows:
\begin{align*}
    \begin{array}[t]{@{}r@{}lr@{\kern.5ex}l@{\kern2ex}lr@{\kern.5ex}l@{\kern2ex}lr@{\kern.5ex}l@{\kern2ex}lr@{}}
        \type{T},\type{U} ::= {} &
        \type{T \times U} & \text{pair}
        & \sepr &
        \type{T \lolli U} & \text{function}
        & \sepr &
        \type{\1} & \text{unit}
        & \sepr &
        \type{S} & \text{session}
        \\
        \type{S} ::= {} &
        \type{{!}T . S} & \text{send}
        & \sepr &
        \type{{?}T . S} & \text{receive}
        & \sepr &
        \type{\oplus \{ i : T \}_{i \in I}} & \text{select}
        & \sepr &
        \type{\& \{ i : T \}_{i \in I}} & \text{branch}
        \\ \sepr~~ &
        \type{\tEnd}
    \end{array}
\end{align*}

\noindent
Session type duality ($\type{\ol{S}}$) is defined as usual; note that only the continuations, and not the messages, of send and receive types are dualized.
\begin{align*}
    \type{\ol{{!}T . S}}
    &= \type{{?}T . \ol{S}}
    &
    \type{\ol{{?}T . S}}
    &= \type{{!}T . \ol{S}}
    &
    \type{\ol{\oplus \{ i : S_i \}_{i \in I}}}
    &= \type{\& \{ i : \ol{S_i} \}_{i \in I}}
    &
    \type{\ol{\& \{ i : S_i \}_{i \in I}}}
    &= \type{\oplus \{ i : \ol{S_i} \}_{i \in I}}
    &
    \type{\ol{\tEnd}}
    &= \type{\tEnd}
\end{align*}

\begin{figure}[t!]
    \def\defaultHypSeparation{\hskip1ex}
    \begin{mathpar}
        \begin{bussproof}[typ-var]
            \bussAx{
                \term{x}:\type{T} \vdashM \term{x}:\type{T}
            }
        \end{bussproof}
        \and
        \begin{bussproof}[typ-abs]
            \bussAssume{
                \type{\Gamma}, \term{x}:\type{T} \vdashM \term{M} :\type{U}
            }
            \bussUn{
                \type{\Gamma} \vdashM \term{\lam x . M} : \type{T \lolli U}
            }
        \end{bussproof}
        \and
        \begin{bussproof}[typ-app]
            \bussAssume{
                \type{\Gamma} \vdashM \term{M}: \type{T \lolli U}
            }
            \bussAssume{
                \type{\Delta} \vdashM \term{N}: \type{T}
            }
            \bussBin{
                \type{\Gamma}, \type{\Delta} \vdashM \term{M~N}: \type{U}
            }
        \end{bussproof}
        \and
        \begin{bussproof}[typ-unit]
            \bussAx{
                \type{\emptyset} \vdashM \term{()}: \type{\1}
            }
        \end{bussproof}
        \and
        \begin{bussproof}[typ-pair]
            \bussAssume{
                \type{\Gamma} \vdashM \term{M}: \type{T}
            }
            \bussAssume{
                \type{\Delta} \vdashM \term{N}: \type{U}
            }
            \bussBin{
                \type{\Gamma}, \type{\Delta} \vdashM \term{(M,N)}: \type{T \times U}
            }
        \end{bussproof}
        \and
        \begin{bussproof}[typ-split]
            \bussAssume{
                \type{\Gamma} \vdashM \term{M}: \type{T \times T'}
            }
            \bussAssume{
                \type{\Delta}, \term{x}:\type{T}, \term{y}:\type{T'} \vdashM \term{N}: \type{U}
            }
            \bussBin{
                \type{\Gamma}, \type{\Delta} \vdashM \term{\tLet (x,y) = M \tIn N}: \type{U}
            }
        \end{bussproof}
        \and
        \begin{bussproof}[typ-new]
            \bussAx{
                \type{\emptyset} \vdashM \term{\tNew}: \type{S \times \ol{S}}
            }
        \end{bussproof}
        \and
        \begin{bussproof}[typ-spawn]
            \bussAssume{
                \type{\Gamma} \vdashM \term{M}: \type{\1}
            }
            \bussAssume{
                \type{\Delta} \vdashM \term{N} : \type{T}
            }
            \bussBin{
                \type{\Gamma} , \type{\Delta} \vdashM \term{\tSpawn M ; N}: \type{T}
            }
        \end{bussproof}
        \and
        \begin{bussproof}[typ-end]
            \bussAssume{
                \type{\Gamma} \vdashM \term{M} : \type{T}
            }
            \bussUn{
                \type{\Gamma} , \term{x}:\type{\tEnd} \vdashM \term{M} : \type{T}
            }
        \end{bussproof}
        \and
        \begin{bussproof}[typ-send]
            \bussAssume{
                \type{\Gamma} \vdashM \term{M}: \type{T}
            }
            \bussAssume{
                \type{\Delta} \vdashM \term{N}: \type{{!}T . S}
            }
            \bussBin{
                \type{\Gamma} ,  \type{\Delta} \vdashM \term{\tSend M \, N}: \type{S}
            }
        \end{bussproof}
        \and
        \begin{bussproof}[typ-recv]
            \bussAssume{
                \type{\Gamma} \vdashM \term{M}: \type{{?}T . S}
            }
            \bussUn{
                \type{\Gamma} \vdashM \term{\tRecv M}: \type{T \times S}
            }
        \end{bussproof}
        \and
        \begin{bussproof}[typ-sel]
            \bussAssume{
                \type{\Gamma} \vdashM \term{M}: \type{\oplus \{ i : S_i \}_{i \in I}}
            }
            \bussAssume{
                j \in I
            }
            \bussBin{
                \type{\Gamma} \vdashM \term{\tSel j\, M}: \type{S_j}
            }
        \end{bussproof}
        \and
        \begin{bussproof}[typ-case]
            \bussAssume{
                \type{\Gamma} \vdashM \term{M}: \type{{\&}\{ i : S_i \}_{i \in I}}
            }
            \bussAssume{
                \forall i \in I.~ \type{\Delta} \vdashM \term{N_i}: \type{S_i \lolli U}
            }
            \bussBin{
                \type{\Gamma}, \type{\Delta} \vdashM \term{\tCase M \tOf \{ i : N_i \}_{i \in I}}: \type{U}
            }
        \end{bussproof}
        \\ \dashes \\
        \begin{bussproof}[typ-buf]
            \bussAx{
                \type{\emptyset} \vdashB \term{\tBfr{\epsilon}}: \type{S'} > \type{S'}
            }
        \end{bussproof}
        \and
        \begin{bussproof}[typ-buf-send]
            \bussAssume{
                \type{\Gamma} \vdashM \term{M}: \type{T}
            }
            \bussAssume{
                \type{\Delta} \vdashB \term{\tBfr{\vec{m}}}: \type{S'} > \type{S}
            }
            \bussBin{
                \type{\Gamma}, \type{\Delta} \vdashB \term{\tBfr{\vec{m},M}}: \type{S'} > \type{{!}T . S}
            }
        \end{bussproof}
        \and
        \begin{bussproof}[typ-buf-sel]
            \bussAssume{
                \type{\Gamma} \vdashB \term{\tBfr{\vec{m}}}: \type{S'} > \type{S_j}
            }
            \bussAssume{
                j \in I
            }
            \bussBin{
                \type{\Gamma} \vdashB \term{\tBfr{\vec{m},j}}: \type{S'} > \type{\oplus \{ i : S_i \}_{i \in I}}
            }
        \end{bussproof}
        \\ \dashes \\
        \begin{bussproof}[typ-main]
            \bussAssume{
                \type{\Gamma} \vdashM \term{M}: \type{\hat{T}}
            }
            \bussUn{
                \type{\Gamma} \vdashC{\tMain} \term{\tMain\, M}: \type{\hat{T}}
            }
        \end{bussproof}
        \and
        \begin{bussproof}[typ-child]
            \bussAssume{
                \type{\Gamma} \vdashM \term{M}: \type{\1}
            }
            \bussUn{
                \type{\Gamma} \vdashC{\tChild} \term{\tChild\, M}: \type{\1}
            }
        \end{bussproof}
        \and
        \begin{bussproof}[typ-par]
            \bussAssume{
                \type{\Gamma} \vdashC{\phi_1} \term{C}: \type{T_1}
            }
            \bussAssume{
                \type{\Delta} \vdashC{\phi_2} \term{D}: \type{T_2}
            }
            \bussAssume{
                \{\type{T_1},\type{T_2}\} \setminus \{\type{T}\} = \{\type{\1}\}
            }
            \bussTern{
                \type{\Gamma}, \type{\Delta} \vdashC{\phi_1 + \phi_2} \term{C \prl D}: \type{T}
            }
        \end{bussproof}
        \and
        \begin{bussproof}[typ-res]
            \bussAssume{
                \type{\Gamma} \vdashB \term{\tBfr{\vec{m}}}: \type{S'} > \type{S}
            }
            \bussAssume{
                \type{\Delta}, \term{x}: \type{S'} \vdashC{\phi} \term{C}: \type{T}
            }
            \bussAssume{
                \type{\Gamma'} , \term{y}:\type{\ol{S}} = \type{\Gamma} , \type{\Delta}
            }
            \bussTern{
                \type{\Gamma'} \vdashC{\phi} \term{\pRes{x\tBfr{\vec{m}}y} C}: \type{T}
            }
        \end{bussproof}
    \end{mathpar}
    \caption{\protect\LAST* typing rules for terms (top), buffers (center), and configurations (bottom).}\label{f:LASTn:LAST:type}
\end{figure}

The type system for \LAST* has three layers: typing for terms, for buffers, and for configurations.
Typing for terms uses judgments of the form 
\[
    \type{\Gamma} \vdashM \term{M} : \type{T}
\]
which decrees that $\term{M}$ has a behavior described by $\type{T}$ using the typing context $\type{\Gamma}$, which is defined as a list of variable-type assignments $\term{x}:\type{T}$.

\Cref{f:LASTn:LAST:type} (top) gives the typing rules for terms.
Rules often combine typing contexts~$\type{\Gamma}$ and $\type{\Delta}$ to form $\type{\Gamma} , \type{\Delta}$; this implicitly assumes that the domains of $\type{\Gamma}$ and $\type{\Delta}$ are disjoint. We briefly comment on them:
\begin{itemize}

	\item
        Rules~\ruleLabel{typ-var}, \ruleLabel{typ-abs}, \ruleLabel{typ-app}, \ruleLabel{typ-unit}, \ruleLabel{typ-pair}, and~\ruleLabel{typ-split} are standard typing rules for functional terms.

    \item
        Rule~\ruleLabel{typ-new} types the $\term{\tNew}$ construct as a pair of dual session types.

    \item
        Rule~\ruleLabel{typ-spawn} takes a term $\term{M}$ of type $\type{\1}$ and a term $\term{N}$ of type $\type{T}$ to type the $\term{\tSpawn}$ construct as $\type{T}$; as we will see in the typing of configurations, child threads must be typed $\type{\1}$, which explains the type of $\term{M}$.

    \item
        Rule~\ruleLabel{typ-end} allows weakening typing contexts with $\type{\tEnd}$-typed variables, as closed sessions are not used.

    \item
        Rule~\ruleLabel{typ-send} takes a term $\term{M}$ of type $\type{T}$ and a term $\term{N}$ of type $\type{{!}T.S}$ to type a $\term{\tSend}$ of $\term{M}$ along $\term{N}$ as the continuation type $\type{S}$.
        Dually, Rule~\ruleLabel{typ-recv} takes a term $\term{M}$ of type $\type{{?}T.S}$ to type a $\term{\tRecv}$ along $\term{M}$ as a pair of the message's payload type and continuation $\type{T \times S}$.

    \item
        Rule~\ruleLabel{typ-sel} takes a term $\term{M}$ of type $\type{\oplus \{ i : S_i \}_{i \in I}}$ to type the selection of some~\mbox{$\term{j} \in I$} along~$\term{M}$ as the continuation $\type{S_j}$.
        Rule~\ruleLabel{typ-case} takes a term $\term{M}$ of type~$\type{\& \{ i : S_i \}_{i \in I}}$ and for every $\term{i} \in I$ a term $\term{N_i}$ typed $\type{S_i \lolli U}$ (i.e., a function from the label~$\term{i}$'s continuation type $\type{S_i}$ to some common but arbitrary type $\type{U}$) to type a $\term{\tCase}$ along~$\term{M}$ as~$\type{U}$.
\end{itemize}

\begin{example}
    \label{x:LASTn:LAST:typTerm}
    To illustrate the typing rules, let us derive the typing of term $\term{\sff{Son}_{s',m'}}$ from \Cref{s:LASTn:example}.
    As in \Cref{x:LASTn:LAST:Shop}, we omit the $\term{\tClose}$ construct.
    To type the sugared terms $\term{\tLet x = M \tIn N}$ we use a sugared Rule~\ruleLabel{typ-let} derivable from Rules~\ruleLabel{typ-abs} and~\ruleLabel{typ-app}.
    We consider $\type{\sff{Str}}$ (string), $\type{\sff{B}}$ (book), and $\type{\sff{P}}$ (payment) to be primitive non-linear types that can be weakened/contracted at will and are self-dual.
    Below, $\type{S} = \type{{!}\sff{P}.{?}\sff{B}.\tEnd}$ and $\type{S'} = \type{{?}\sff{B}.\tEnd}$.
    \begin{align}
        & \begin{bussproof}
            \def\ScoreOverhang{1pt}
            \def\defaultHypSeparation{\hskip2mm}
            \bussAx[\ruleLabel{typ-var}]{
                \term{m'_1}:\type{{?}\sff{B}.\tEnd} \vdashM \term{m'_1}:\type{{?}\sff{B}.\tEnd}
            }
            \bussUn[\ruleLabel{typ-recv}]{
                \term{m'_1}:\type{{?}\sff{B}.\tEnd} \vdashM \term{\tRecv m'_1}:\type{\sff{Book} \times \tEnd}
            }
            \bussAx[\ruleLabel{typ-var}]{
                \term{\textit{book}}:\type{\sff{B}} \vdashM \term{\textit{book}} : \type{\sff{B}}
            }
            \bussUn[\ruleLabel{typ-end}]{
                \term{\textit{book}}:\type{\sff{B}} , \term{m'_2}:\type{\tEnd} \vdashM \term{\textit{book}} : \type{\sff{B}}
            }
            \bussBin[\ruleLabel{typ-split}]{
                \term{m'_1}:\type{{?}\sff{B}.\tEnd} \vdashM \term{\tLet (\textit{book},m'_2) = \tRecv m'_1 \tIn \textit{book}} : \type{\sff{B}}
            }
        \end{bussproof}
        \label{eq:LASTn:LAST:type}
        \\[5pt]
        & \begin{bussproof}
            \bussAx[\ruleLabel{typ-var}]{
                \term{s'_2}:\type{S} \vdashM \term{s'_2} : \type{S}
            }
            \bussAx[\ruleLabel{typ-var}]{
                \term{m'}:\type{{!}S.{?}\sff{B}.\tEnd} \vdashM \term{m'} : \type{{!}S.{?}\sff{B}.\tEnd}
            }
            \bussBin[\ruleLabel{typ-send}]{
                \term{s'_2}:\type{S} , \term{m'}:\type{{!}S.{?}\sff{B}.\tEnd} \vdashM \term{\tSend s'_2 \, m'} : \type{{?}\sff{B}.\tEnd}
            }
            \bussAssume{
                \eqref{eq:LASTn:LAST:type}
            }
            \bussBin[\ruleLabel{typ-let}]{
                \term{s'_2}:\type{S} , \term{m'}:\type{{!}S.{?}\sff{B}.\tEnd} \vdashM \term{\tLet m'_1 = \tSend s'_2 \, m' \tIn \ldots} : \type{\sff{B}}
            }
            \doubleLine
            \bussUn[\dots]{
                \term{s'}:\type{{!}\sff{Str}.\oplus\{\sff{buy}:S,\sff{blurb}:S'\}} , \term{m'}:\type{{!}S.{?}\sff{B}.\tEnd} \vdashM \term{\sff{Son}_{s',m'}} : \type{\sff{B}}
            }
        \end{bussproof}
        \nonumber
    \end{align}

    Similarly, the typings of the shop and the mother are as follows:
    \begin{align*}
        \term{s}:\type{{?}\sff{Str}.\&\{\sff{buy}:\ol{S},\sff{blurb}:\ol{S'}\}} &\vdashM \term{\sff{Shop}_s} : \type{\1}
        \\
        \term{m}:\type{{?}S.{!}\sff{B}.\tEnd} &\vdashM \term{\sff{Mother}_m} : \type{\1}
    \end{align*}
    As such, the types of $\term{s},\term{s'}$ and $\term{m},\term{m'}$ are pairwise dual.
    Hence, the entire system is typed simply
    \[
        \type{\emptyset} \vdashM \term{\sff{Sys}} : \type{\sff{B}}.
        \tag*{\qedhere}
    \]
\end{example}

The typing for  buffered channels  uses judgments of the form $$\type{\Gamma} \vdashB \term{\tBfr{\vec{m}}} : \type{S'}>\type{S}$$
The difference between $\type{S'}$ and $\type{S}$ is determined by the messages in the buffer $\term{\vec{m}}$: $\type{S}$ denotes a sequence of sends and selections corresponding to the values and labels in $\term{\vec{m}}$, after which the type continues as $\type{S'}$.
Thus, $\type{S}$ denotes the type of a buffered channel's output endpoint before it sent the messages in $\term{\vec{m}}$, and $\type{S'}$ denotes this endpoint's current type.
This way, $\type{\ol{S}}$ signifies the type of the buffered channel's input endpoint, corresponding to a sequence of receives and offers corresponding to the values and labels in $\term{\vec{m}}$.

\Cref{f:LASTn:LAST:type} (center) gives the three typing rules for buffers.
The typing context $\type{\Gamma}$ is used in the typing of the values in the buffer.
Rule~\ruleLabel{typ-buf} types an empty buffer; as such,~$\type{S'} = \type{S}$.
Rule~\ruleLabel{typ-buf-send} takes a value $\term{v}$ of type $\type{T}$ and a buffer typed $\type{S'}>\type{S}$ to insert~$\term{v}$ at the start of the buffer now typed $\type{S'}>\type{{!}T.S}$.
Rule~\ruleLabel{typ-buf-sel} takes a buffer typed~$\type{S'}>\type{S_j}$ (for some~\mbox{$\term{j} \in I$}) to insert $\term{j}$ at the start of the buffer, now typed as~$\type{S'}>\type{\oplus \{ i : S_i \}_{i \in I}}$.

Typing for configurations uses judgments of the form $$\type{\Gamma} \vdashC{\phi} \term{C} : \type{T}$$ where the  marker $\term{\phi}$  indicates whether $\term{C}$ contains the main thread ($\term{\phi} = \term{\tMain}$) or only child threads ($\term{\phi} = \term{\tChild}$).
When composing  configurations marked $\term{\phi_1}$ and $\term{\phi_2}$, we  compute a new marker $\term{\phi_1 + \phi_2}$, as follows:
\begin{mathpar}
    \term{\tMain + \tChild} \deq \term{\tChild + \tMain} \deq \term{\tMain}
    \and
    \term{\tChild + \tChild} \deq \term{\tChild}
    \and
    \text{($\term{\tMain + \tMain}$ is undefined)}
\end{mathpar}
\Cref{f:LASTn:LAST:type} (bottom) gives the typing rules for configurations.
Rule~\ruleLabel{typ-main} takes a term of non-session type (denoted $\type{\hat{T}}$) and turns it into a main thread.
Rule~\ruleLabel{typ-child} takes a term of type $\type{\1}$ and turns it into a child thread.
Rule~\ruleLabel{typ-par} composes two configurations typed $\type{T_1}$ and $\type{T_2}$ in parallel; the rule requires one of the configurations to be typed $\type{\1}$ and uses the other configuration's type to type the composition.
Note that if both configurations contain a main thread,  the sum of their markers is undefined, and the rule cannot be applied.
Rule~\ruleLabel{typ-res} types a configuration under a buffered channel with output endpoint $\term{x}$ and input endpoint $\term{y}$.
The rule uses typing for buffers to type the buffer $\type{S'}>\type{S}$.
As such, the configuration should use $\term{x}$ of type $\type{S'}$.
Since $\type{S}$ is the type of $\term{x}$ before it sent the messages already in the buffer, $\term{y}$ should be of type $\type{\ol{S}}$.
Note that $\term{y}$ may be used in the configuration, but may also appear in a message in the buffer.

\subsection{Towards a Faithful Translation of \texorpdfstring{\protect\LAST*}{LAST*} into \APCP}
\label{s:LASTn:LAST:trans}

As already discussed, we are interested in faithfully translating \LAST* into \APCP.
We first discuss what we mean precisely by `faithful'.
Then, we consider some existing translations of (variants of) the \lamcalc with call-by-value semantics into (variants of) the \picalc; we briefly discuss their status with respect to our definition of faithfulness.

\paragraph{Faithfulness.}
We shall aim for correct translations in the sense of Gorla~\cite{journal/ic/Gorla10}, a well-established set of criteria whereby faithful translations should satisfy \emph{operational correspondence}, a criterion that is divided into
\emph{completeness} and \emph{soundness} properties.
Intuitively, completeness says that the target language can do \emph{enough} to represent the behavior of the source language, whereas soundness says that the target language does not do \emph{too much}.
Alternatively, one may see completeness and soundness as properties that ensure that reduction steps are \emph{preserved} and \emph{reflected} by the translation, respectively.

We formulate these requirements in the specific setting of \LAST* and \APCP, by considering a translation of configurations into processes, denoted `$\encc{z}{\cdot}$', where, as usual, $z$ denotes a name on which the resulting process exhibits the behavior of the source term:
\begin{description}

    \item[Completeness.]
        If $\term{C} \reddC \term{D}$, then $\encc{z}{C} \redd* \encc{z}{D}$.

    \item[Soundness.]
        If $\encc{z}{C} \redd* Q$, then there exists $\term{D}$ such that $\term{C} \reddC^\ast \term{D}$ and $Q \redd* \encc{z}{D}$.

\end{description}

\paragraph{Existing Translations are not Sound.}
We are not aware of translations of \lamcalci with call-by-value semantics into the \picalc that satisfy soundness as stated above.
To illustrate the problem, let us consider the well-known translation by Milner~\cite{conf/icalp/Milner90,journal/mscs/Milner92}. 
This is a translation into a \picalc with synchronous communication, while we seek a translation into \APCP, which is asynchronous. This discrepancy is not crucial: the unsound reductions we will show next are not enabled by the asynchrony of \APCP (i.e., they are also enabled in synchronous processes).

It actually suffices to consider Milner's translation of variables and applications.
By adapting this to \APCP, we obtain the translations given next.
Below, we  write `$\_$' to denote a fresh name of type $\bullet$\,; when sending names denoted `$\_$', we omit binders~`$\pRes{\_\_}$'.
\begin{align*}
    \encc{z}{x} &\deq \pRes{ab} ( \pOut z[\_,a] \| \pIn b(\_,c) ; \pOut x[\_,c] )
    \\
    \encc{z}{M~N} &\deq \pRes{ab} \pRes{cd} \big( \pIn a(\_,e) ; \pRes{fg} ( \pOut e[\_,f] \| \pIn c(\_,h) ; \pOut g[h,z] ) \| \encc{b}{M} \| \encc{d}{N} \big)
\end{align*}
Consider the very simple term $\term{x~y}$, which cannot reduce.
Yet, its translation $\encc{z}{x~y}$ does reduce, as shown next  (each step underlines the send and receive that synchronize):
\begin{align*}
    \encc{z}{x~y} &= \pRes{a_1b_1} \pRes{c_1d_1} \big(
        \begin{array}[t]{@{}l@{}}
            \underline{\pIn a_1(\_,e_1)} ; \pRes{f_1g_1} ( \pOut e_1[\_,f_1] \| \pIn c_1(\_,h_1) ; \pOut g_1[h_1,z] )
            \\
            {} \| \pRes{a_2b_2} ( \underline{\pOut b_1[\_,a_2]} \| \pIn b_2(\_,c_2) ; \pOut x[\_,c_2] )
            \\
            {} \| \pRes{a_3b_3} ( \pOut d_1[\_,a_3] \| \pIn b_3(\_,c_3) ; \pOut y[\_,c_3] )
        \big)
    \end{array}
    \\
    &\redd \pRes{c_1d_1} \pRes{a_2b_2} \big(
        \begin{array}[t]{@{}l@{}}
            \pRes{f_1g_1} ( \underline{\pOut a_2[\_,f_1]} \| \pIn c_1(\_,h_1) ; \pOut g_1[h_1,z] )
            \\
            {} \| \underline{\pIn b_2(\_,c_2)} ; \pOut x[\_,c_2]
            \\
            {} \| \pRes{a_3b_3} ( \pOut d_1[\_,a_3] \| \pIn b_3(\_,c_3) ; \pOut y[\_,c_3] )
        \big)
    \end{array}
    \\
    &\redd \pRes{c_1d_1} \pRes{f_1g_1} \big(
        \begin{array}[t]{@{}l@{}}
            \underline{\pIn c_1(\_,h_1)} ; \pOut g_1[h_1,z]
            \\
            {} \| \pOut x[\_,f_1]
            \\
            {} \| \pRes{a_3b_3} ( \underline{\pOut d_1[\_,a_3]} \| \pIn b_3(\_,c_3) ; \pOut y[\_,c_3] )
        \big)
    \end{array}
    \\
    &\redd \pRes{f_1g_1} \pRes{a_3b_3} ( \pOut g_1[a_3,z] \| \pOut x[\_,f_1] \| \pIn b_3(\_,c_3) ; \pOut y[\_,c_3] ) \nredd
\end{align*}
It is clear that this final term, which cannot reduce any further, cannot be reconciled with any source \LAST* term.
Hence, the translation is unsound.

Besides Milner's, other translations of \lamcalci with call-by-value semantics into (variants of) the \picalc have been proposed by Lindley and Morris'~\cite{conf/esop/LindleyM15} (a precise formalization of Walder's~\cite{conf/icfp/Wadler12,journal/jfp/Wadler14}) and by Vasconcelos'~\cite{report/Vasconcelos00}.
The most significant difference between them is how they translate variables: Milner uses sends; Lindley and Morris use forwarders; Vasconcelos uses substitutions.
None of them satisfy soundness as defined above, although Vasconcelos' enjoys a weaker form of soundness up to barbed congruence.

By examining these three translations, we observe that the call-by-value semantics is overly contextual, in the sense that determining whether a subterm may reduce depends on the context.
This way, e.g., Rule~\ruleLabel{red-lam} (\Cref{f:LASTn:LAST:terms}) only applies to abstractions applied to values.
The semantic rules for the \picalc are much less contextual,
so translations require additional machinery to prevent unwanted reductions.
We are not aware of  translations of call-by-value \lamcalci into $\pi$ that are sound, which in our view suggests that such a sound translation may not exist.

Based on this discussion, we conclude that the call-by-value semantics of \LAST* does not lend itself for a faithful representation in \APCP.
To address this issue, in the next section we will propose a variant of \LAST with a call-by-name semantics.
This variant, denoted \LASTn, will admit a faithful (i.e., sound and complete)   translation into \APCP.

\section{\LASTn and a Faithful Translation into \APCP}
\label{s:LASTn:LASTn}

Here we present \LASTn, a call-by-name variant of \LAST.
Our presentation proceeds gradually, based on the presentation we gave for  \LAST* (\Cref{s:LASTn:LAST}).
We then give a typed translation of \LASTn into \APCP, and prove that it is faithful as discussed in \Cref{s:LASTn:LAST:trans}.
Finally, we show how the translation can help us identify a fragment of deadlock-free \LASTn programs.

\subsection{The Language of \LASTn}
\label{s:LASTn:LASTn:language}

To define \LASTn, we modify the semantics of \LAST*.
We only need to change the definitions of reduction contexts and term reduction in \Cref{f:LASTn:LAST:terms}.
The crux is that call-by-name semantics applies abstractions more eagerly, and not only when arguments are values.
\ifappendix For the reader's convenience, \Cref{as:LASTn:lang} contains a full, self-contained definition of the language and type system of \LASTn.\fi

To be precise, term reduction changes Rules~\ruleLabel{red-lam} and~\ruleLabel{red-pair} as follows:
\begin{mathpar}
    \begin{bussproof}[red-lam]
        \bussAx{
            \term{(\lam x . M)~\underline{N}} \reddM \term{M \{ \underline{N}/x \}}
        }
    \end{bussproof}
    \and
    \begin{bussproof}[red-pair]
        \bussAx{
            \term{\tLet (x,y) = (\underline{N_1},\underline{N_2}) \tIn M} \reddM \term{M \{ \underline{N_1}/x,\underline{N_2}/y \}}
        }
    \end{bussproof}
\end{mathpar}
Above, we underline modified parts: instead of requiring values, the rules allow arbitrary terms.
Moreover, to enforce the eager application of these rules, we remove from the definition of reduction contexts the clauses $\term{v~\tCtx{R}}$, $\term{(\tCtx{R},M)}$, and $\term{(M,\tCtx{R})}$.
That is, \LASTn disallows terms in parameter positions to reduce.

\begin{example}
    \label{x:LASTn:CbN}
    We consider the term from \Cref{x:LASTn:CbV}, now using the call-by-name semantics of \LASTn:
    \[
        \term{\big(\lam x . x~(\lam y . y)\big)~\big((\lam w . w)~(\lam z . z)\big)}
        \reddM
        \term{\big((\lam w . w)~(\lam z . z)\big)~(\lam y . y)}
        \reddM
        \term{(\lam z . z)~(\lam y . y)}
        \reddM
        \term{\lam y . y}
    \]
    Notice that the function on $\term{x}$ is applied \emph{before} evaluating its parameter as in \Cref{x:LASTn:CbV}.
\end{example}

As is to be expected, there are subtle differences between the semantics of \LASTn and \APCP.
In order to deal with these discrepancies, while enabling the desired completeness and soundness results, we enrich \LASTn with \emph{explicit substitutions} and \emph{closed sessions}.

\paragraph{Explicit substitutions.}

Term reductions~\ruleLabel{red-lam} and~\ruleLabel{red-pair} substitute variables for terms, regardless of the position of these variables.
To mimic this in \APCP, we would have to be able to substitute the translations of variables for translations of terms.
Although this is possible when those variables occur in reduction contexts, there is no mechanism for representing this in \APCP when they occur in different contexts (where their translations are blocked by prefixes).

One way of dealing with this issue is to equate translations up to substitutions, using the so-called  \emph{substitution lifting}~\cite{conf/fossacs/ToninhoCP12}.
Here, we opt for an alternative, more direct treatment based on \emph{explicit substitutions} (see, e.g., \cite{conf/fsttcs/LevyM99}).
Intuitively, explicit substitutions delay variable substitution until those variables occur in reduction contexts.
To incorporate explicit substitutions, we proceed as follows:
\begin{itemize}

    \item
        We add explicit substitutions to terms $\term{M \tSub{ N/x }}$ and configurations $\term{C \tSub{ N/x }}$.
        These are not meant to be used when writing programs, instead appearing and disappearing as programs reduce ({runtime} syntax).
        That is, they appear when abstractions are applied and pairs deconstructed:
        \begin{mathpar}
            \begin{bussproof}[red-lam]
                \bussAx{
                    \term{(\lam x . M)~N} \reddM \term{M \tSub{ N/x }}
                }
            \end{bussproof}
            \and
            \begin{bussproof}[red-pair]
                \bussAx{
                    \term{\tLet (x,y) = (N_1,N_2) \tIn M} \reddM \term{M \tSub{ N_1/x,N_2/y }}
                }
            \end{bussproof}
        \end{mathpar}
        Explicit substitutions disappear when they meet the substituted variable:
        \[
            \begin{bussproof}[red-subst]
                \bussAx{
                    \term{x \tSub{ N/x }} \reddM \term{N}
                }
            \end{bussproof}
        \]

    \item
        We add reduction contexts $\term{\tCtx{R} \tSub{ N/x }}$ and configuration contexts $\term{\tCtx{G} \tSub{ N/x }}$.
        We then define a structural congruence for terms, denoted $\equivM$, with a single rule that enables extruding the scope of explicit substitutions across reduction contexts:
        \[
            \begin{bussproof}[sc-sub-ext]
                \bussAssume{
                    \term{x} \notin \fv(\tCtx{R})
                }
                \bussUn{
                    \term{(\tCtx{R}[M]) \tSub{ N/x }} \equivM \tCtx{R}[M \tSub{ N/x }]
                }
            \end{bussproof}
        \]
        We close term reduction under this structural congruence:
        \[
            \begin{bussproof}[red-lift-sc]
                \bussAssume{
                    \term{M} \equivM \term{M'}
                }
                \bussAssume{
                    \term{M'} \reddM \term{N'}
                }
                \bussAssume{
                    \term{N'} \equivM \term{N}
                }
                \bussTern{
                    \term{M} \reddM \term{N}
                }
            \end{bussproof}
        \]
        We also add rules to the structural congruence for configurations that lift explicit substitutions in terms and enable scope extrusion on the level of configurations, respectively:
        \begin{mathpar}
            \begin{bussproof}[sc-conf-sub]
                \bussAx{
                    \term{\phi \, ( M \tSub{ N/x } )} \equivC \term{(\phi \, M) \tSub{ N/x }}
                }
            \end{bussproof}
            \and
            \begin{bussproof}[sc-conf-sub-ext]
                \bussAssume{
                    \term{x} \notin \fv(\tCtx{G})
                }
                \bussUn{
                    \term{(\tCtx{G}[C]) \tSub{ N/x }} \equivC \tCtx{G}[C \tSub{ N/x }]
                }
            \end{bussproof}
        \end{mathpar}

    \item
        When moving terms between threads, we need to make sure that we do not affect variables that are bound by explicit substitutions.
        As such we define a specific form of thread context, denoted~$\tCtx{\hat{F}}$, that disallows the hole to appear under explicit substitutions.
        Configuration reductions~\ruleLabel{red-spawn}, \ruleLabel{red-send}, and~\ruleLabel{red-recv} (cf.\ \Cref{f:LASTn:LAST:confs}) then use these specific contexts.

    \item
        To type explicit substitutions, we add the following rules for terms and configurations:
        \begin{mathpar}
            \begin{bussproof}[typ-sub]
                \bussAssume{
                    \type{\Gamma} , \term{x}:\type{T} \vdashM \term{M} : \type{U}
                }
                \bussAssume{
                    \type{\Delta} \vdashM \term{N} : \type{T}
                }
                \bussBin{
                    \type{\Gamma} , \type{\Delta} \vdashM \term{M \tSub{ N/x }} : \type{U}
                }
            \end{bussproof}
            \and
            \begin{bussproof}[typ-conf-sub]
                \bussAssume{
                    \type{\Gamma} , \term{x}:\type{T} \vdashC{\phi} \term{C} : \type{U}
                }
                \bussAssume{
                    \type{\Delta} \vdashM \term{N} : \type{T}
                }
                \bussBin{
                    \type{\Gamma} , \type{\Delta} \vdashC{\phi} \term{C \tSub{ N/x }} : \type{U}
                }
            \end{bussproof}
        \end{mathpar}

\end{itemize}

\begin{example}
    \label{x:LASTn:explSub}
    We revisit \Cref{x:LASTn:CbN}, this time using explicit substitutions:
    \begin{align*}
        \term{\big(\lam x . x~(\lam y . y)\big)~\big((\lam w . w)~(\lam z . z)\big)}
        &\reddM
        \term{\big(x~(\lam y . y)\big) \tSub{ \big((\lam w . w)~(\lam z . z)\big)/x }}
        \\
        &\equivM
        \term{(x \tSub{ \big((\lam w . w)~(\lam z . z)\big)/x })~(\lam y . y)}
        \\
        &\reddM
        \term{\big((\lam w . w)~(\lam z . z)\big)~(\lam y . y)}
        \\
        &\reddM
        \term{(w \tSub{ (\lam z . z)/w })~(\lam y . y)}
        \reddM
        \term{(\lam z . z)~(\lam y . y)}
        \\
        &\reddM
        \term{z \tSub{(\lam y . y)/z}}
        \reddM
        \term{\lam y . y}
        \tag*{\qedhere}
    \end{align*}
\end{example}

\begin{figure}[p]
    \begin{align*}
        \term{\tMain \, \sff{Sys}}
        &= \term{\tMain \, \tLet (s,s') = \tNew \tIn \ldots}
        \\
        &\reddC \term{\pRes{y\tBfr{\epsi}y'} \tMain \, \tLet (s,s') = (y,y') \tIn \ldots}
        \\
        &\reddC \term{\pRes{y\tBfr{\epsi}y'} \tMain \, \tSpawn \sff{Shop}_s ; \ldots \tSub{y/s,y'/s'}}
        \\
        &\equivC \term{\pRes{y\tBfr{\epsi}y'} \big( ( \tMain \, \tSpawn \sff{Shop}_s ; \ldots ) \tSub{y/s,y'/s'} \big)}
        \\
        &\reddC \term{\pRes{y\tBfr{\epsi}y'} \big( ( \tMain \, \tLet (m,m') = \tNew \tIn \ldots \prl \tChild \, \sff{Shop}_s ) \tSub{y/s,y'/s'} \big)}
        \\
        &\equivC \term{\pRes{y\tBfr{\epsi}y'} \big( ( \begin{array}[t]{@{}l@{}}
            \tMain \, \tLet (m,m') = \tNew \tIn \ldots 
            \\
            \prl \tChild \, \tLet (\textit{title},s_1) = \tRecv (s \tSub{y/s}) \tIn \ldots ) \tSub{y'/s'} \big)
        \end{array} }
        \\
        &\reddC \term{\pRes{y\tBfr{\epsi}y'} \big( ( \tMain \, \tLet (m,m') = \tNew \tIn \ldots \prl \tChild \, \tLet (\textit{title},s_1) = \tRecv y \tIn \ldots ) \tSub{y'/s'} \big)}
        \\
        &\reddC^2 \term{\pRes{y\tBfr{\epsi}y'} \big( \pRes{z\tBfr{\epsi}z'} ( \tMain \, \tSpawn \sff{Mother}_m ; \sff{Son}_{s',m'} ) \tSub{z/m,z'/m',y'/s'} \prl \tChild \, \sff{Shop}_y \big)}
        \\
        &\reddC \term{\pRes{y\tBfr{\epsi}y'} \big( \pRes{z\tBfr{\epsi}z'} \big(
            \tMain \, \sff{Son}_{s',m'} \tSub{z'/m',y'/s'} 
            \prl \tChild \, \sff{Mother}_m \tSub{z/m} \big) \prl \tChild \, \sff{Shop}_y \big)
        }
        \\
        &\equivC \term{\pRes{y\tBfr{\epsi}y'} \big( \pRes{z\tBfr{\epsi}z'} \big( \begin{array}[t]{@{}l@{}}
            \tMain \, \sff{Son}_{s',m'} \tSub{z'/m',y'/s'} 
            \\
            \prl \tChild \, \tLet (x,m_1) = \tRecv (m \tSub{z/m}) \tIn \ldots \big) \prl \tChild \, \sff{Shop}_y \big)
        \end{array} }
        \\
        &\reddC \term{\pRes{y\tBfr{\epsi}y'} \big( \pRes{z\tBfr{\epsi}z'} \big( \begin{array}[t]{@{}l@{}}
            \tMain \, \sff{Son}_{s',m'} \tSub{z'/m',y'/s'} 
            \\
            \prl \tChild \, \tLet (x,m_1) = \tRecv z \tIn \ldots \big) \prl \tChild \, \sff{Shop}_y \big)
        \end{array} }
        \\
        &= \term{\pRes{y\tBfr{\epsi}y'} \big( \pRes{z\tBfr{\epsi}z'} \big( \tMain \, \sff{Son}_{s',m'} \tSub{z'/m',y'/s'} \prl \tChild \, \sff{Mother}_z \big) \prl \tChild \, \sff{Shop}_y \big)} =: \term{\sff{Sys}^1}
    \end{align*}
    \begin{align*}
        & \term{\sff{Son}_{s',m'} \tSub{z'/m',y'/s'}}
        \\
        &= \term{\tLet s'_1 = \tSend \text{``Dune''} \, s' \tIn \ldots \tSub{z'/m',y'/s'}}
        \\
        &\reddM \term{\tLet s'_2 = \tSel \sff{buy} \, s'_1 \tIn \ldots \tSub{\tSend \text{``Dune''} \, s'/s'_1,z'/m',y'/s'}}
        \\
        &\reddM \term{\tLet m'_1 = \tSend s'_2 \, (m' \tSub{z'/m'}) \tIn \ldots \tSub{\tSel \sff{buy} \, s'_1/s'_2,\tSend \text{``Dune''} \, s'/s'_1,y'/s'}}
        \\
        &\reddM \term{\tLet m'_1 = \tSend s'_2 \, z' \tIn \ldots \tSub{\tSel \sff{buy} \, s'_1/s'_2,\tSend \text{``Dune''} \, s'/s'_1,y'/s'}}
        \\
        &\reddM \term{ \begin{array}[t]{@{}l@{}}
            \tLet (\textit{book},m'_2) = \tRecv (m'_1 \tSub{\tSend s'_2 \, z'/m'_1}) 
            \\
            \quad \tIn \ldots \tSub{\tSel \sff{buy} \, s'_1/s'_2,\tSend \text{``Dune''} \, s'/s'_1,y'/s'}
        \end{array} }
        \\
        &\reddM \term{\tLet (\textit{book},m'_2) = \tRecv (\tSend s'_2 \, z') \tIn \ldots \tSub{\tSel \sff{buy} \, s'_1/s'_2,\tSend \text{``Dune''} \, s'/s'_1,y'/s'}}
        \\
        &=: \term{\sff{Son}_{s',z'}^1}
    \end{align*}

    \caption{Behavior of the bookshop scenario in \LASTn, part~1/2 (cf.\ \Cref{x:LASTn:LASTn:Shop}).}\label{f:LASTn:LASTn:bookshop1}
\end{figure}
    
\begin{figure}[p]
    \begin{align*}
        \term{\sff{Sys}^1}
        &\reddC^5 \term{\pRes{y\tBfr{\epsi}y'} \big( \pRes{z\tBfr{\epsi}z'} \big( \tMain \, \sff{Son}^1_{s',z'} \prl \tChild \, \sff{Mother}_z \big) \prl \tChild \, \sff{Shop}_y \big)}
        \\
        &\equivC \term{
            \pRes{y\tBfr{\epsi}y'} \big(
                \begin{array}[t]{@{}l@{}}
                    \pRes{z'\tBfr{\epsi}z} \big( \tMain \, \tLet (\textit{book},m'_2) = \tRecv (\tSend s'_2 \, z') \tIn \ldots \prl \tChild \, \sff{Mother}_z \big)
                    \\
                    \tSub{\tSel \sff{buy} \, s'_1/s'_2,\tSend \text{``Dune''} \, s'/s'_1,y'/s'} \prl \tChild \, \sff{Shop}_y
                \big)
            \end{array}
        }
        \\
        &\reddC \term{
            \pRes{y\tBfr{\epsi}y'} \big(
                \begin{array}[t]{@{}l@{}}
                    \pRes{z'\tBfr{s'_2}z} \big( \begin{array}[t]{@{}l@{}}
                        \tMain \, \tLet (\textit{book},m'_2) = \tRecv z' \tIn \ldots 
                        \\
                        \prl \tChild \, \tLet (x,m_1) = \tRecv z \tIn \ldots \big)
                    \end{array}
                    \\
                    \tSub{\tSel \sff{buy} \, s'_1/s'_2,\tSend \text{``Dune''} \, s'/s'_1,y'/s'} \prl \tChild \, \sff{Shop}_y
                \big)
            \end{array}
        }
        \\
        &\reddC \term{
            \pRes{y\tBfr{\epsi}y'} \big(
                \pRes{z'\tBfr{\epsi}z} \big(
                    \begin{array}[t]{@{}l@{}}
                        \tMain \, \tLet (\textit{book},m'_2) = \tRecv z' \tIn \ldots
                        \\
                        \prl \tChild \, \begin{array}[t]{@{}l@{}}
                            \tLet (x,m_1) = (s'_2,z) \tIn \ldots 
                            \\
                            \tSub{\tSel \sff{buy} \, s'_1/s'_2,\tSend \text{``Dune''} \, s'/s'_1,y'/s'}
                        \big)
                        \prl \tChild \, \sff{Shop}_y
                    \big)
                \end{array}
            \end{array}
        }
        \\
        &\reddC \term{
            \pRes{y\tBfr{\epsi}y'} \big(
                \pRes{z'\tBfr{\epsi}z} \big(
                    \begin{array}[t]{@{}l@{}}
                        \tMain \, \tLet (\textit{book},m'_2) = \tRecv z' \tIn \ldots
                        \\
                        \prl \tChild \, \begin{array}[t]{@{}l@{}}
                            \tLet x_1 = \tSend \text{visa} \, (x \lBrace \begin{array}[t]{@{}l@{}}
                                s'_2/x,\tSel \sff{buy} \, s'_1/s'_2,
                                \\
                                \tSend \text{``Dune''} \, s'/s'_1,y'/s')
                            \rBrace \end{array}
                            \\
                            \tIn \ldots \tSub{z/m_1}
                        \big)
                        \prl \tChild \, \sff{Shop}_y
                    \big)
                \end{array}
            \end{array}
        }
        \\
        &\reddC \term{
            \pRes{y\tBfr{\epsi}y'} \big(
                \pRes{z'\tBfr{\epsi}z} \big(
                    \begin{array}[t]{@{}l@{}}
                        \tMain \, \tLet (\textit{book},m'_2) = \tRecv z' \tIn \ldots
                        \\
                        \prl \tChild \, \begin{array}[t]{@{}l@{}}
                            \tLet x_1 = \tSend \text{visa} \, (s'_2 \lBrace \begin{array}[t]{@{}l@{}}
                                \tSel \sff{buy} \, s'_1/s'_2,
                                \\
                                \tSend \text{``Dune''} \, s'/s'_1,y'/s') \rBrace
                            \end{array}
                            \\
                            \tIn \ldots \tSub{z/m_1}
                        \big)
                        \prl \tChild \, \sff{Shop}_y
                    \big)
                \end{array}
            \end{array}
        }
        \\
        &\reddC \term{
            \pRes{y\tBfr{\epsi}y'} \big(
                \pRes{z'\tBfr{\epsi}z} \big(
                    \begin{array}[t]{@{}l@{}}
                        \tMain \, \tLet (\textit{book},m'_2) = \tRecv z' \tIn \ldots
                        \\
                        \prl \tChild \,
                        \begin{array}[t]{@{}l@{}}
                            \tLet x_1 = \tSend \text{visa} \, (\tSel \sff{buy} \, (s'_1 \lBrace \begin{array}[t]{@{}l@{}}
                                \tSend \text{``Dune''} \, s'/s'_1,
                                \\
                                y'/s' \rBrace ))
                            \end{array}
                            \\
                            \tIn \ldots \tSub{z/m_1}
                        \big)
                        \prl \tChild \, \sff{Shop}_y
                    \big)
                \end{array}
            \end{array}
        }
        \\
        &\reddC \term{
            \pRes{y\tBfr{\epsi}y'} \big(
                \pRes{z'\tBfr{\epsi}z} \big(
                    \begin{array}[t]{@{}l@{}}
                        \tMain \, \tLet (\textit{book},m'_2) = \tRecv z' \tIn \ldots
                        \\
                        \prl \tChild \,
                        \begin{array}[t]{@{}l@{}}
                            \tLet x_1 = \tSend \text{visa} \, (\tSel \sff{buy} \, (\tSend \text{``Dune''} \, (s' \tSub{\!y'/s'\!})))
                            \\
                            \tIn \ldots \tSub{z/m_1}
                        \big)
                        \prl \tChild \, \sff{Shop}_y
                    \big)
                \end{array}
            \end{array}
        }
        \\
        &\reddC \term{
            \pRes{y\tBfr{\epsi}y'} \big(
                \pRes{z'\tBfr{\epsi}z} \big(
                    \begin{array}[t]{@{}l@{}}
                        \tMain \, \tLet (\textit{book},m'_2) = \tRecv z' \tIn \ldots
                        \\
                        \prl \tChild \,
                        \begin{array}[t]{@{}l@{}}
                            \tLet x_1 = \tSend \text{visa} \, (\tSel \sff{buy} \, (\tSend \text{``Dune''} \, y'))
                            \\
                            \tIn \ldots \tSub{z/m_1}
                        \big)
                        \prl \tChild \, \sff{Shop}_y
                    \big) \mathrel{\normalcolor=:} \sff{Sys}^2
                \end{array}
            \end{array}
        }
    \end{align*}

    \caption{Behavior of the bookshop scenario in \LASTn, part~2/2 (cf.\ \Cref{x:LASTn:LASTn:Shop}).}\label{f:LASTn:LASTn:bookshop2}
\end{figure}

\begin{example}[The Bookshop Scenario in \LASTn]
    \label{x:LASTn:LASTn:Shop}
    Recall the bookshop scenario introduced in \Cref{s:LASTn:example}, already illustrated by \Cref{x:LASTn:LAST:Shop} for \LAST*.
    Although the end result is the same, the CbN semantics and explicit substitutions of \LASTn do change the behavior of the system ($\term{\tMain \, \sff{Sys}}$) compared to its behavior under \LAST*.

    We start illustrating this behavior by reconsidering how the system starts by setting up channels; as shown in \Cref{f:LASTn:LASTn:bookshop1} (top), not all explicit substitutions can be immediately resolved.
    At this point, the explicit substitutions on the main thread cannot be resolved, as neither of $\term{m'},\term{s'}$ appears under a reduction context in $\term{\sff{Son}_{s',m'}}$.
    We continue the example by examining the behavior of the son at this point, shown in \Cref{f:LASTn:LASTn:bookshop1} (bottom).
    This chain of explicit substitutions cannot be resolved further: $\term{s'_2}$ does not appear under a reduction context.
    In fact, they will not be resolved until the son has delegated his (pending) access to $\term{z'}$ to his mother, along with his pending messages in the form of explicit substitutions, as shown by the sequence leading to process $\sff{Sys}^2$ in \Cref{f:LASTn:LASTn:bookshop2}. 

    Next, the mother's thread finally performs her son's two original outputs:
    \begin{align*}
        \term{\sff{Sys}^2}
        &\reddC^3 \term{
            \pRes{y'\tBfr{\text{visa},\sff{buy},\text{``Dune''}}y} \big(
                \pRes{z'\tBfr{\epsi}z} \big(
                    \begin{array}[t]{@{}l@{}}
                        \tMain \, \tLet (\textit{book},m'_2) = \tRecv z' \tIn \ldots
                        \\
                        \prl \tChild \,
                        \tLet x_1 = y' \tIn \ldots \tSub{z/m_1}
                    \big)
                    \prl \tChild \, \sff{Shop}_y
                \big)
            \end{array}
        }
    \end{align*}

    At this point, the system behaves roughly as in \Cref{x:LASTn:LAST:Shop}: there is no further session delegation (as we just witnessed between the son and the mother), so upcoming explicit substitutions can be resolved straightforwardly.
\end{example}

\paragraph{Closed sessions.}

Both in \LASTn and \APCP, closed sessions have no associated behavior.
This entails some management on both sides to garbage-collect unused variables with type~$\type{\tEnd}$.
As a result, translating terms with variables requires two forms: one for regular variables, and another for variables typed~$\type{\tEnd}$.
Because such a translation can get burdensome, we opt to add a form of `semi-explicit' closing to \LASTn.
The idea is that when the type of a buffered channel's endpoint is $\type{\tEnd}$ that endpoint may be closed.
Once both endpoints of a channel have been closed, the channel and its buffer can be garbage-collected.
This treatment ensures relieves us from translating unused variables, in such a way that the translation of terms with variables comes in one intuitive form only.
To formalize this idea, we proceed as follows:
\begin{itemize}

    \item
        We add terms for endpoint closing $\term{\tClose M ; N}$ and a reduction context $\term{\tClose \tCtx{R} ; N}$, with behavior on the level of configurations defined by the following new reduction:
        \[
            \begin{bussproof}[red-close]
                \bussAx{
                    \term{\pRes{x\tBfr{\vec{m}}y} ( \tCtx{F}[\tClose x ; M] \prl C )} \reddC \term{\pRes{\tNil\tBfr{\vec{m}}y} ( \tCtx{F}[M] \prl C )}
                }
            \end{bussproof}
        \]
        Here, `$\term{\tNil}$' is a special runtime endpoint variable that indicates that the endpoint has been closed and hence that the variable is not used anywhere inside the enclosed configuration.
        Then, buffered channels can be garbage-collected as follows:
        \[
            \begin{bussproof}[red-res-nil]
                \bussAx{
                    \term{\pRes{\tNil\tBfr{\epsi}\tNil} C} \reddC \term{C}
                }
            \end{bussproof}
        \]

    \item
        We type the new session closing construct by replacing Rule~\ruleLabel{typ-end} as follows:
        \[
            \begin{bussproof}[typ-close]
                \bussAssume{
                    \type{\Gamma} \vdashM \term{M} : \type{\tEnd}
                }
                \bussAssume{
                    \type{\Delta} \vdashM \term{N} : \type{T}
                }
                \bussBin{
                    \type{\Gamma} , \type{\Delta} \vdashM \term{\tClose M ; N} : \type{T}
                }
            \end{bussproof}
        \]

    \item
        We add a new session type for endpoints that have been  already closed.
        With a slight abuse of notation, we denote it as $\type{\tNil}$.
        We disallow using this type in the typing of terms, leaving it only to appear in the typing of buffered channels.
        To ensure that closed endpoint variables are not used in configurations, we add side-conditions to Rule~\ruleLabel{typ-res}: $\term{x} = \term{\tNil}$ if and only if $\type{S'} = \type{\tNil}$ and similarly for $\term{y}$.
        The following two new rules enable typing partially closed buffered channels:
        \begin{mathpar}
            \begin{bussproof}[typ-buf-end-L]
                \bussAx{
                    \type{\emptyset} \vdashB \term{\tBfr{\epsi}} : \type{\tEnd}>\type{\tNil}
                }
            \end{bussproof}
            \and
            \begin{bussproof}[typ-buf-end-R]
                \bussAx{
                    \type{\emptyset} \vdashB \term{\tBfr{\epsi}} : \type{\tNil}>\type{\tEnd}
                }
            \end{bussproof}
        \end{mathpar}

\end{itemize}

\begin{example}
    We illustrate the semi-explicit closing of sessions in \LASTn in the context of  the bookstore scenario introduced in \Cref{s:LASTn:example}.
    Consider the system where all session interactions have taken place, and all three threads are ready to close their sessions:
    \[
        \term{\sff{Sys}}
        \reddC^\ast \term{\pRes{y\tBfr{\epsi}y'} \big( \pRes{z\tBfr{\epsi}z'} \big( \tMain \, \tClose z' ; \text{book}(\text{``Dune''}) \prl \tChild \, \tClose z ; \tClose y' \big) \prl \tChild \tClose y \big)} =: \term{\sff{Sys}^c}
    \]
    The order in which endpoints are closed is unimportant; here, we execute the $\term{\tClose}$ primitives from left to right.
    \begin{align*}
        \term{\sff{Sys}^c}
        &\equivC \term{\pRes{y\tBfr{\epsi}y'} \big( \pRes{z'\tBfr{\epsi}z} \big( \tMain \, \tClose z' ; \text{book}(\text{``Dune''}) \prl \tChild \, \tClose z ; \tClose y' \big) \prl \tChild \tClose y \big)}
        \\
        &\reddC \term{\tMain \, \text{book}(\text{``Dune''}) \prl \pRes{y\tBfr{\epsi}y'} \big( \pRes{\tNil\tBfr{\epsi}z} \tChild \, \tClose z ; \tClose y' \prl \tChild \tClose y \big)}
        \\
        &\equivC \term{\tMain \, \text{book}(\text{``Dune''}) \prl \pRes{y\tBfr{\epsi}y'} \big( \pRes{z\tBfr{\epsi}\tNil} \tChild \, \tClose z ; \tClose y' \prl \tChild \tClose y \big)}
        \\
        &\reddC \term{\tMain \, \text{book}(\text{``Dune''}) \prl \pRes{y\tBfr{\epsi}y'} \big( \pRes{\tNil\tBfr{\epsi}\tNil} \tChild \, \tClose y' \prl \tChild \tClose y \big)}
        \\
        &\reddC \term{\tMain \, \text{book}(\text{``Dune''}) \prl \pRes{y\tBfr{\epsi}y'} \big( \tChild \, \tClose y' \prl \tChild \tClose y \big)}
        \\
        &\equivC \term{\tMain \, \text{book}(\text{``Dune''}) \prl \pRes{y'\tBfr{\epsi}y} \big( \tChild \, \tClose y' \prl \tChild \tClose y \big)}
        \\
        &\reddC \term{\tMain \, \text{book}(\text{``Dune''}) \prl \tChild \, () \prl \pRes{\tNil\tBfr{\epsi}y} \tChild \tClose y}
        \\
        &\reddC \term{\tMain \, \text{book}(\text{``Dune''}) \prl \pRes{\tNil\tBfr{\epsi}y} \tChild \tClose y}
        \\
        &\equivC \term{\tMain \, \text{book}(\text{``Dune''}) \prl \pRes{y\tBfr{\epsi}\tNil} \tChild \tClose y}
        \tag{$\ast$}
        \\
        &\reddC \term{\tMain \, \text{book}(\text{``Dune''}) \prl \pRes{\tNil\tBfr{\epsi}\tNil} \tChild ()}
        \\
        &\reddC \term{\tMain \, \text{book}(\text{``Dune''}) \prl \tChild ()}
        \\
        &\reddC \term{\tMain \, \text{book}(\text{``Dune''})}
    \end{align*}

    Additionally, we illustrate the typing of half-closed sessions on the configuration marked~($\ast$).
    Recall from \Cref{x:LASTn:LAST:typTerm} that we consider $\type{\sff{B}}$ (book) a primitive non-linear type that can be weakened/contracted at will and is self-dual.
    We have:
    \begin{align}
        & \begin{bussproof}
            \def\defaultHypSeparation{\hskip1mm}
            \bussAx[\ruleLabel{typ-var}]{
                \term{y}:\type{\tEnd} \vdashM \term{y} : \type{\tEnd}
            }
            \bussAx[\ruleLabel{typ-unit}]{
                \type{\emptyset} \vdashM \term{()} : \type{\1}
            }
            \bussBin[\ruleLabel{typ-close}]{
                \term{y}:\type{\tEnd} \vdashM \term{\tClose y ; ()} : \type{\1}
            }
            \bussUn[\ruleLabel{typ-child}]{
                \term{y}:\type{\tEnd} \vdashC{\tChild} \term{\tChild \, \tClose y ; ()} : \type{\1}
            }
            \bussAx[\ruleLabel{typ-buf-end-L}]{
                \type{\emptyset} \vdashB \term{\tBfr{\epsi}} : \type{\tEnd} > \type{\tNil}
            }
            \bussBin[\ruleLabel{typ-res}]{
                \type{\emptyset} \vdashC{\tChild} \term{\pRes{y\tBfr{\epsi}\tNil} \tChild \, \tClose y ; ()} : \type{\1}
            }
        \end{bussproof}
        \label{eq:LASTn:closing}
        \\
        & \begin{bussproof}
            \bussAx{
                \type{\emptyset} \vdashM \term{\text{book}(\text{``Dune''})} : \type{\sff{B}}
            }
            \bussUn[\ruleLabel{typ-main}]{
                \type{\emptyset} \vdashC{\tMain} \term{\tMain \, \text{book}(\text{``Dune''})} : \type{\sff{B}}
            }
            \bussAssume{
                \eqref{eq:LASTn:closing}
            }
            \bussBin[\ruleLabel{typ-par}]{
                \type{\emptyset} \vdashC{\tMain} \term{\tMain \, \text{book}(\text{``Dune''}) \prl \pRes{y\tBfr{\epsi}\tNil} \tChild \, \tClose y ; ()} : \type{\sff{B}}
            }
        \end{bussproof}
        \nonumber
        \tag*{\qedhere}
    \end{align}
\end{example}

\paragraph{Type preservation}

\LASTn satisfies the expected correctness properties: protocol fidelity and communication safety.
Both properties follow directly from \emph{type preservation}.

\begin{restatable}[Type Preservation for \protect\LASTn]{theorem}{tLASTnTp}
\label{t:LASTn:LASTn:tp}
    Given $\type{\Gamma} \vdashC{\phi} \term{C} : \type{T}$, if $\term{C} \equivC \term{D}$ or $\term{C} \reddC \term{D}$, then $\type{\Gamma} \vdashC{\phi} \term{D} : \type{T}$.
\end{restatable}

\begin{proof}
    By proving subject congruence ($\equivC$ preserves typing) and subject reduction ($\reddC$ preserves typing) separately.
    In both cases, we first prove separately subject congruence and subject reduction on the term level ($\equivM$ and $\reddM$, respectively).
    These results follow by induction on the derivation of the structural congruence or reduction.
    Below we give a few interesting cases\ifappendix; \Cref{as:LASTn:tp} contains detailed proofs\fi.
    \begin{itemize}

        \item
            For $\equivM$, the only interesting case is the only base case, i.e.,%
            \global\def\proofScTermsSubExt{%
                Rule~\ruleLabel{sc-sub-ext}:
                \[
                    \term{x} \notin \fv(\tCtx{R}) \implies \term{(\tCtx{R}[M]) \tSub{ N/x }} \equivM \tCtx{R}[M \tSub{ N/x }]
                \]

                We apply induction on the structure of the reduction context $\tCtx{R}$.
                As an interesting, representative case, consider $\tCtx{R} = \term{L \tSub{ \tCtx{R'}/y }}$.
                Assuming $\term{x} \notin \fv(\tCtx{R})$, we have $\term{x} \notin \fv(\term{L}) \cup \fv(\tCtx{R'})$.
                We apply inversion of typing:
                \[
                    \begin{bussproof}
                        \bussAssume{
                            \type{\Gamma}, \term{y}:\type{U} \vdashM \term{L}: \type{T}
                        }
                        \bussAssume{
                            \type{\Delta}, \term{x}:\type{U'} \vdashM \tCtx{R'}[M]: \type{U}
                        }
                        \bussBin[\ruleLabel{typ-sub}]{
                            \type{\Gamma}, \type{\Delta}, \term{x}:\type{U'} \vdashM \term{L \tSub{ (\tCtx{R'}[M])/y }}: \type{T}
                        }
                        \bussAssume{
                            \type{\Delta'} \vdashM \term{N}: \type{U'}
                        }
                        \bussBin[\ruleLabel{typ-sub}]{
                            \type{\Gamma}, \type{\Delta}, \type{\Delta'} \vdashM \term{(L \tSub{ (\tCtx{R'}[M])/y }) \tSub{ N/x }}: \type{T}
                        }
                    \end{bussproof}
                \]
                We can derive:
                \[
                    \begin{bussproof}
                        \bussAssume{
                            \type{\Delta}, \term{x}:\type{U'} \vdashM \tCtx{R'}[M]: \type{U}
                        }
                        \bussAssume{
                            \type{\Delta'} \vdashM \term{N}: \type{U'}
                        }
                        \bussBin[\ruleLabel{typ-sub}]{
                            \type{\Delta}, \type{\Delta'} \vdashM \term{(\tCtx{R'}[M]) \tSub{ N/x }}: \type{U}
                        }
                    \end{bussproof}
                \]
                Since $\term{x} \notin \fv(\tCtx{R'})$, by Rule~\ruleLabel{sc-sub-ext}, $\term{(\tCtx{R'}[M]) \tSub{ N/x }} \equivM \tCtx{R'}[M \tSub{ N/x }]$.
                Then, by the IH, $\type{\Delta}, \type{\Delta'} \vdashM \tCtx{R'}[M \tSub{ N/x }]: \type{U}$.
                Hence, we can conclude the following:
                \[
                    \begin{bussproof}
                        \bussAssume{
                            \type{\Gamma}, \term{y}:\type{U} \vdashM \term{L}: \type{T}
                        }
                        \bussAssume{
                            \type{\Delta}, \type{\Delta'} \vdashM \tCtx{R'}[M \tSub{ N/x }]: \type{U}
                        }
                        \bussBin[\ruleLabel{typ-sub}]{
                            \type{\Gamma}, \type{\Delta}, \type{\Delta'} \vdashM \term{L \tSub{ (\tCtx{R'}[M \tSub{ N/x }])/y }}: \type{T}
                        }
                    \end{bussproof}
                \]
                It is straightforward to see that this reasoning works in opposite direction as well.%
            }
            \proofScTermsSubExt

        \item
            For $\reddM$, all cases are straightforward.

        \item
            For $\equivC$, we detail the base case%
            \global\def\proofScResSwap{%
                Rule~\ruleLabel{sc-res-swap}: $\term{\pRes{x\tBfr{\epsi}y} C} \equivC \term{\pRes{y\tBfr{\epsi}x} C}$.

                The analysis depends on whether exactly one of $\term{x},\term{y}$ is $\term{\tNil}$ or not.
                We discuss both cases.
                \begin{itemize}

                    \item
                        Exactly one of $\term{x},\term{y}$ is $\term{\tNil}$; w.l.o.g., assume $\term{x} = \term{\tNil}$.
                        We have the following:
                        \begin{align*}
                            & \begin{bussproof}
                                \bussAx[\ruleLabel{typ-buf-end-R}]{
                                    \type{\emptyset} \vdashB \term{\tBfr{\epsi}} : \type{\tNil}>\type{\tEnd}
                                }
                                \bussAssume{
                                    \type{\Gamma} , \term{x} : \type{\tEnd} \vdashC{\phi} \term{C} : \type{T}
                                }
                                \bussBin[\ruleLabel{typ-res}]{
                                    \type{\Gamma} \vdashC{\phi} \term{\pRes{\tNil\tBfr{\epsi}y} C} : \type{T}
                                }
                            \end{bussproof}
                            \\
                            & \equivC
                            \\
                            & \begin{bussproof}
                                \bussAx[\ruleLabel{typ-buf-end-L}]{
                                    \type{\emptyset} \vdashB \term{\tBfr{\epsi}} : \type{\tEnd}>\type{\tNil}
                                }
                                \bussAssume{
                                    \type{\Gamma} , \term{x} : \type{\tEnd} \vdashC{\phi} \term{C} : \type{T}
                                }
                                \bussBin[\ruleLabel{typ-res}]{
                                    \type{\Gamma} \vdashC{\phi} \term{\pRes{y\tBfr{\epsi}\tNil} C} : \type{T}
                                }
                            \end{bussproof}
                        \end{align*}

                    \item
                        Neither or both of $\term{x},\term{y}$ are $\term{\tNil}$.
                        We have the following:
                        \begin{align*}
                            & \begin{bussproof}
                                \bussAx[\ruleLabel{typ-buf}]{
                                    \type{\emptyset} \vdashB \term{\tBfr{\epsi}} : \type{S'}>\type{S'}
                                }
                                \bussAssume{
                                    \type{\Gamma} , \term{x}:\type{S'} , \term{y}:\type{\ol{S'}} \vdashC{\phi} \term{C} : \type{T}
                                }
                                \bussBin[\ruleLabel{typ-res}]{
                                    \type{\Gamma} \vdashC{\phi} \term{\pRes{x\tBfr{\epsi}y} C} : \type{T}
                                }
                            \end{bussproof}
                            \\
                            & \equivC
                            \\
                            & \begin{bussproof}
                                \bussAx[\ruleLabel{typ-buf}]{
                                    \type{\emptyset} \vdashB \term{\tBfr{\epsi}} : \type{\ol{S'}}>\type{\ol{S'}}
                                }
                                \bussAssume{
                                    \type{\Gamma} , \term{x}:\type{S'} , \term{y}:\type{\ol{S'}} \vdashC{\phi} \term{C} : \type{T}
                                }
                                \bussBin[\ruleLabel{typ-res}]{
                                    \type{\Gamma} \vdashC{\phi} \term{\pRes{y\tBfr{\epsi}x} C} : \type{T}
                                }
                            \end{bussproof}
                        \end{align*}

                \end{itemize}%
            }
            \proofScResSwap

        \item
            For $\reddC$, we detail the base case%
            \global\def\proofSrSend#1{%
                Rule~\ruleLabel{red-send}: $\term{\pRes{x\tBfr{\vec{m}}y}(\tCtx{\hat{F}}[\tSend M \, x] \prl C)} \reddC \term{\pRes{x\tBfr{M,\vec{m}}y}(\tCtx{\hat{F}}[x] \prl C)}$.

                This case follows by induction on the structure of $\tCtx{\hat{F}}$.
                The inductive cases follow from the IH straightforwardly.
                The fact that the hole in $\tCtx{\hat{F}}$ does not occur under an explicit substitution guarantees that we can move $\term{M}$ out of the context of $\tCtx{\hat{F}}$ and into the buffer.
                We consider the base case ($\tCtx{\hat{F}} = \term{\phi\,\tCtx{R}}$).
                By well-typedness, it must be that $\tCtx{R} = \tCtx{R_1}[\tClose \tCtx{R_2} ; M]$.
                We apply induction on the structures of $\tCtx{R_1},\tCtx{R_2}$ and consider the base cases: $\tCtx{R_1} = \tCtx{R_2} = \tHole$.
                We apply inversion of typing, w.l.o.g.\ assuming that $\term{\phi} = \term{\tMain}$ and $\term{y} \in \fv(\term{C})$:
                \begin{align}
                    &\begin{bussproof}
                        \bussAssume{
                            \type{\Delta_1} \vdashM \term{M}: \type{T}
                        }
                        \bussAx[\ruleLabel{typ-var}]{
                            \term{x}:\type{{!}T . \tEnd} \vdashM \term{x}: \type{{!}T . \tEnd}
                        }
                        \bussBin[\ruleLabel{typ-send}]{
                            \type{\Delta_1}, \term{x}:\type{{!}T . \tEnd} \vdashM \term{\tSend M \, x}: \type{\tEnd}
                        }
                        \bussAssume{
                            \type{\Delta_2} \vdashM \term{N} : \type{U}
                        }
                        \bussBin[\ruleLabel{typ-close}]{
                            \type{\Delta_1} , \type{\Delta_2} , \term{x}:\type{{!}T . \tEnd} \vdashM \term{\tClose (\tSend M \, x) ; N} : \type{U}
                        }
                        \bussUn[\ruleLabel{typ-main}]{
                            \type{\Delta_1}, \type{\Delta_2} , \term{x}:\type{{!}T . \tEnd} \vdashC{\tMain} \term{\tMain\, (\tClose (\tSend M \, x) ; N)}: \type{U}
                        }
                    \end{bussproof}
                    \label{eq:LASTn:tp:#1red-send1}
                    \\
                    &\begin{bussproof}
                        \bussAssume{
                            \eqref{eq:LASTn:tp:#1red-send1}
                        }
                        \bussAssume{
                            \type{\Lambda}, \term{y}:\type{\ol{S}} \vdashC{\tChild} \term{C}: \type{\1}
                        }
                        \bussBin[\ruleLabel{typ-par-R}]{
                            \type{\Delta_1} , \type{\Delta_2} , \type{\Lambda}, \term{x}:\type{{!}T . \tEnd}, \term{y}:\type{\ol{S}} \vdashC{\tMain} \term{\tMain\, (\tClose (\tSend M \, x) ; N) \prl C}: \type{U}
                        }
                    \end{bussproof}
                    \label{eq:LASTn:tp:#1red-send2}
                    \\
                    &\begin{bussproof}
                        \bussAssume{
                            \type{\Gamma} \vdashB \term{\tBfr{\vec{m}}}: \type{{!}T . \tEnd} > \type{S}
                        }
                        \bussAssume{
                            \eqref{eq:LASTn:tp:#1red-send2}
                        }
                        \bussBin[\ruleLabel{typ-res}]{
                            \type{\Gamma}, \type{\Delta_1} , \type{\Delta_2}, \type{\Lambda} \vdashC{\tMain} \term{\pRes{x\tBfr{\vec{m}}y}(\tMain\, (\tClose (\tSend M \, x) ; N) \prl C)}: \type{U}
                        }
                    \end{bussproof}
                    \nonumber
                \end{align}

                Note that the derivation of $\type{\Gamma} \vdashB \term{\tBfr{\vec{m}}}: \type{{!}T . \tEnd} > \type{S}$ depends on the size of $\term{\vec{m}}$.
                By induction on the size of $\term{\vec{m}}$ (\ih{2}), we derive $\type{\Gamma}, \type{\Delta_1} \vdashB \term{\tBfr{M,\vec{m}}}: \type{\tEnd} > \type{S}$:
                \begin{itemize}

                    \item
                        If $\term{\vec{m}}$ is empty, it follows by inversion of typing that $\type{\Gamma} = \type{\emptyset}$ and $\type{S} = \type{{!}T . \tEnd}$:
                        \[
                            \begin{bussproof}
                                \bussAx[\ruleLabel{typ-buf}]{
                                    \type{\emptyset} \vdashB \term{\tBfr{\epsi}}: \type{{!}T . \tEnd} > \type{{!}T . \tEnd}
                                }
                            \end{bussproof}
                        \]
                        Then, we derive the following:
                        \[
                            \begin{bussproof}
                                \bussAssume{
                                    \type{\Delta_1} \vdashM \term{M}: \type{T}
                                }
                                \bussAx[\ruleLabel{typ-buf}]{
                                    \type{\emptyset} \vdashB \term{\tBfr{\epsi}}: \type{\tEnd} > \type{\tEnd}
                                }
                                \bussBin[\ruleLabel{typ-buf-send}]{
                                    \type{\Delta_1} \vdashB \term{\tBfr{M}}: \type{\tEnd} > \type{{!}T . \tEnd}
                                }
                            \end{bussproof}
                        \]

                    \item
                        If $\term{\vec{m}} = \term{\vec{m}',L}$, it follows by inversion of typing that $\type{\Gamma} = \type{\Gamma'},\type{\Gamma''}$ and $\type{S} = \type{{!}T' . \tEnd'}$:
                        \[
                            \begin{bussproof}
                                \bussAssume{
                                    \type{\Gamma'} \vdashM \term{L}: \type{T'}
                                }
                                \bussAssume{
                                    \type{\Gamma''} \vdashB \term{\tBfr{\vec{m}'}}: \type{{!}T . \tEnd} > \type{\tEnd'}
                                }
                                \bussBin[\ruleLabel{typ-buf-send}]{
                                    \type{\Gamma'}, \type{\Gamma''} \vdashB \term{\tBfr{\vec{m}',L}}: \type{{!}T . \tEnd} > \type{{!}T' . \tEnd'}
                                }
                            \end{bussproof}
                        \]
                        By \ih{2}, $\type{\Gamma''}, \type{\Delta_1} \vdashB \term{\tBfr{M,\vec{m}'}}: \type{\tEnd} > \type{\tEnd'}$, allowing us to derive the following:
                        \[
                            \begin{bussproof}
                                \bussAssume{
                                    \type{\Gamma'} \vdashM \term{L}: \type{T'}
                                }
                                \bussAssume{
                                    \type{\Gamma''}, \type{\Delta_1} \vdashB \term{\tBfr{M,\vec{m}'}}: \type{\tEnd} > \type{\tEnd'}
                                }
                                \bussBin[\ruleLabel{typ-buf-send}]{
                                    \type{\Gamma'}, \type{\Gamma''}, \type{\Delta_1} \vdashB \term{\tBfr{M,\vec{m}',L}}: \type{\tEnd} > \type{{!}T' . \tEnd'}
                                }
                            \end{bussproof}
                        \]

                    \item
                        If $\term{\vec{m}} = \term{\vec{m}', j}$, it follows by inversion of typing that there exist types $\type{S_i}$ for each $i$ in a set of labels $I$, where $j \in I$, such that $\type{S} = \type{\oplus \{ i : S_i \}_{i \in I}}$:
                        \[
                            \begin{bussproof}
                                \bussAssume{
                                    \type{\Gamma} \vdashB \term{\tBfr{\vec{m}'}}: \type{{!}T . \tEnd} > \type{S_j}
                                }
                                \bussUn[\ruleLabel{typ-buf-sel}]{
                                    \type{\Gamma} \vdashB \term{\tBfr{\vec{m},j}}: \type{{!}T . \tEnd} > \type{\oplus \{ i : S_i \}_{i \in I}}
                                }
                            \end{bussproof}
                        \]
                        By \ih{2}, $\type{\Gamma}, \type{\Delta_1} \vdashB \term{\tBfr{M,\vec{m}'}}: \type{\tEnd} > \type{S_j}$, so we derive the following:
                        \[
                            \begin{bussproof}
                                \bussAssume{
                                    \type{\Gamma}, \type{\Delta_1} \vdashB \term{\tBfr{M,\vec{m}'}}: \type{\tEnd} > \type{S_j}
                                }
                                \bussUn[\ruleLabel{typ-buf-sel}]{
                                    \type{\Gamma}, \type{\Delta_1} \vdashB \term{\tBfr{M,\vec{m}',j}}: \type{\tEnd} > \type{\oplus \{ i : S_i \}_{i \in I}}
                                }
                            \end{bussproof}
                        \]

                \end{itemize}

                Now, we can derive the typing of the structurally congruent configuration:
                \begin{align}
                    &\begin{bussproof}
                        \bussAx[\ruleLabel{typ-var}]{
                            \term{x}:\type{\tEnd} \vdashM \term{x}: \type{\tEnd}
                        }
                        \bussAssume{
                            \type{\Delta_2} \vdashM \term{N} : \type{U}
                        }
                        \bussBin[\ruleLabel{typ-close}]{
                            \type{\Delta_2} , \term{x}:\type{\tEnd} \vdashM \term{\tClose x ; N} : \type{U}
                        }
                        \bussUn[\ruleLabel{typ-main}]{
                            \type{\Delta_2} , \term{x}:\type{\tEnd} \vdashC{\tMain} \term{\tMain\, (\tClose x ; N)}: \type{U}
                        }
                    \end{bussproof}
                    \label{eq:LASTn:tp:#1red-send3}
                    \\
                    &\begin{bussproof}
                        \def\ScoreOverhang{2pt}
                        \bussAssume{
                            \type{\Gamma}, \type{\Delta_1} \vdashB \term{\tBfr{M,\vec{m}}}: \type{\tEnd} > \type{S}
                        }
                        \bussAssume{
                            \eqref{eq:LASTn:tp:#1red-send3}
                        }
                        \bussAssume{
                            \type{\Lambda}, \term{y}:\type{\ol{S}} \vdashC{\tChild} \term{C}: \type{\1}
                        }
                        \bussBin[\ruleLabel{typ-par-R}]{
                            \type{\Delta_2} , \type{\Lambda}, \term{x}:\type{\tEnd}, \term{y}:\type{\ol{S}} \vdashC{\tMain} \term{\tMain\, (\tClose x ; N) \prl C}: \type{U}
                        }
                        \bussBin[\ruleLabel{typ-res}]{
                            \type{\Gamma}, \type{\Delta_1} , \type{\Delta_2} , \type{\Lambda} \vdashC{\tMain} \term{\pRes{x\tBfr{M,\vec{m}}y} ( \tMain\, (\tClose x ; N) \prl C )}: \type{U}
                        }
                    \end{bussproof}
                    \nonumber
                \end{align}%
            }
            \proofSrSend{}
            \qedhere

    \end{itemize}
\end{proof}


\subsection{Faithfully Translating \LASTn into \APCP}
\label{s:LASTn:LASTn:trans}

We now define a typed translation from \LASTn into \APCP.
Then, we show that this translation is operationally complete and sound, in the sense of \Cref{s:LASTn:LAST:trans}.

\paragraph{The translation: Key ideas.}

Before we define the translation, let us discuss its key idea, inspired by Milner's translation of the \emph{lazy} \lamcalc~\cite{journal/mscs/Milner92}.
The most important design decision is how to translate variables.
Variables serve two purposes: (i)~as placeholders for future substitutions and (ii)~as an access point to buffered channels.
As such, the translation uses variables as sends (i)~to trigger explicit substitutions or (ii)~to trigger interaction with a buffer.
Another important design decision is the translation of buffers: the buffer's types guide the translation to form a sequence of inputs/outputs that are explicitly forwarded (i.e., not using the forwarder process) between the buffer endpoints.

\paragraph{The translation: Types.}

Equipped with these ideas about translating variables and buffers, let us give a deeper insight into our translation by taking a look at how it turns functional and session types in \LASTn into session types in \APCP.

Before moving on, note that \LASTn does not guarantee deadlock-freedom by typing.
As such, there is no point in annotating types with priorities, as the priority requirements induced by typing translated processes may not always be satisfiable.
We therefore use a special \APCP typing judgment, denoted $\vdash*$, which indicates well-typedness modulo priority annotations and requirements.
We will recover priorities and deadlock-freedom in \Cref{s:LASTn:LASTn:df}.

\Cref{d:LASTn:LASTn:transTypes} below introduces two forms of type translation: while we use $\enct{-}$ to translate the types of terms/configurations, we use $\enct*{-}$ to translate the types of variables (remember that variables serve a special role, so a special translation should not be surprising).
We write $\enct*{\Gamma}$ to denote a component-wise translation, where each assignment $\term{x}:\type{T}$ translates to an assignment $x:\enct*{T}$.
As such, our \emph{typed} translation takes, e.g., a typed term $\type{\Gamma} \vdashM \term{M} : \type{T}$ and returns a typed process 
\[
    \vdash* \encc{z}{M} \typInf \enct*{\Gamma} , z:\enct{T}.
\]
Above, $\encc{z}{M}$ is the process that models the behavior of $\term{M}$ on a fresh name $z$ (defined next).

Let us now define the two forms of translation of types (that are mutually recursive):
\begin{definition}[Translation of Types]
    \label{d:LASTn:LASTn:transTypes}
    \begin{align*}
        \enct*{T} &\deq \bullet \tensor \ol{\enct{T}} \quad \text{(if $\type{T} \neq \type{\tNil}$)}
        \span\span
        \\
        \enct{T \times U} &\deq \ol{\enct*{T}} \tensor \ol{\enct*{U}}
        &
        \enct{T \lolli U} &\deq \enct*{T} \parr \enct{U}
        &
        \enct{\1} &\deq \bullet
        \\
        \enct{{!}T.S} &\deq \bullet \tensor \enct*{T} \parr \ol{\enct*{S}}
        &
        \enct{\oplus \{ i : S_i \}_{i \in I}} &\deq \bullet \tensor \& \{ i : \ol{\enct*{S_i}} \}_{i \in I}
        &
        \enct{\tEnd} &\deq \bullet \tensor \bullet
        \\
        \enct{{?}T.S} &\deq \ol{\enct*{T}} \tensor \ol{\enct*{S}}
        &
        \enct{\& \{ i : S_i \}_{i \in I}} &\deq \oplus \{ i : \ol{\enct*{S_i}} \}_{i \in I}
        &
        \enct{\tNil} &\deq \enct*{\tNil} \deq \bullet
    \end{align*}
\end{definition}

Some intuitions follow:
\begin{itemize}

    \item
        The translation $\enct*{-}$ (row~1) clearly shows the translation of variables as sends.

    \item
        The translation $\enct{-}$ translates functional types straightforwardly (row~2): pairs send access to their components as variables, abstractions receive their parameter as a variable and then provide the return type, and units have no behavior.

    \item 
        The translation of session types (rows~3 and~4) is more interesting: at a first glance they seem to be translated dually.
        This is because these are types that belong to buffered channel endpoints.
        As such, when a variable is typed $\type{{!}T.S}$ and we send a term on this variable, the translation sends something along the variable and thus the variable receives.
        Additionally, the translation of sends and selections have an extra send: this signifies a handshake between the variable and the buffer, indicating that both are ready to perform the actual send/selection.

\end{itemize}

\begin{example}
    \label{x:LASTn:LASTn:transTypes}
    We illustrate the translation of types by means of an example.
    After giving the translation of the type $\type{(T \times S) \lolli {!}T.S}$, we break the resulting \APCP type down and explain it in terms of the associated behavior of a process translated from a term implementing the type.
    \begin{align*}
        \enct{(T \times S) \lolli {!}T.S}
        &= \enct*{T \times S} \parr \enct{{!}T.S}
        \\
        &= ( \bullet \tensor \ol{\enct{T \times S}} ) \parr \bullet \tensor \enct*{T} \parr \ol{\enct*{S}}
        \\
        &= ( \bullet \tensor \ol{\ol{\enct*{T}} \tensor \ol{\enct*{S}}} ) \parr \bullet \tensor ( \bullet \tensor \ol{\enct{T}} ) \parr \ol{\bullet \tensor \ol{\enct{S}}}
        \\
        &= ( \bullet \tensor \enct*{T} \parr \enct*{S} ) \parr \bullet \tensor ( \bullet \tensor \ol{\enct{T}} ) \parr \bullet \parr \enct{S}
        \\
        &= ( \bullet \tensor ( \bullet \tensor \ol{\enct{T}} ) \parr \bullet \tensor \ol{\enct{S}} ) \parr \bullet \tensor ( \bullet \tensor \ol{\enct{T}} ) \parr \bullet \parr \enct{S}
    \end{align*}
    This type can be explained as follows:
    \begin{align*}
        & ( \bullet \tensor {}
        &&
        \tikzmark{x:LASTn:transTypes:1s}
        \tikzmark{x:LASTn:transTypes:11s}
        \tikzmark{x:LASTn:transTypes:11e}
        \tag*{announce pair ready}
        \\
        & \q[1] ( \bullet \tensor {}
        &&
        \tikzmark{x:LASTn:transTypes:12s}
        \tikzmark{x:LASTn:transTypes:121s}
        \tikzmark{x:LASTn:transTypes:121e}
        \tag*{announce substitution ready}
        \\
        & \q[2] \ol{\enct{T}}
        &&
        \tikzmark{x:LASTn:transTypes:122s}
        \tikzmark{x:LASTn:transTypes:122e}
        \tag*{actual first component}
        \\
        & \q[1] ) \parr {}
        &&
        \tikzmark{x:LASTn:transTypes:12e}
        \tag*{receive first component}
        \\
        & \q[1] \bullet \tensor {}
        &&
        \tikzmark{x:LASTn:transTypes:13s}
        \tikzmark{x:LASTn:transTypes:13e}
        \tag*{announce substitution ready}
        \\
        & \q[1] \ol{\enct{S}}
        &&
        \tikzmark{x:LASTn:transTypes:14s}
        \tikzmark{x:LASTn:transTypes:14e}
        \tag*{actual second component}
        \\
        & ) \parr {}
        &&
        \tikzmark{x:LASTn:transTypes:1e}
        \tag*{receive parameter}
        \\
        & \bullet \tensor {}
        &&
        \tikzmark{x:LASTn:transTypes:2s}
        \tikzmark{x:LASTn:transTypes:2e}
        \tag*{trigger buffer}
        \\
        & ( \bullet \tensor {}
        &&
        \tikzmark{x:LASTn:transTypes:3s}
        \tikzmark{x:LASTn:transTypes:31s}
        \tikzmark{x:LASTn:transTypes:31e}
        \tag*{announce substitution ready}
        \\
        & \q[1] \ol{\enct{T}}
        &&
        \tikzmark{x:LASTn:transTypes:32s}
        \tikzmark{x:LASTn:transTypes:32e}
        \tag*{actual payload}
        \\
        & ) \parr {}
        &&
        \tikzmark{x:LASTn:transTypes:3e}
        \tag*{receive payload provider}
        \\
        & \bullet \parr {}
        &&
        \tikzmark{x:LASTn:transTypes:4s}
        \tikzmark{x:LASTn:transTypes:4e}
        \tag*{await substitution ready}
        \\
        & \enct{S}
        &&
        \tikzmark{x:LASTn:transTypes:5s}
        \tikzmark{x:LASTn:transTypes:5e}
        \tag*{actual continuation}
    \end{align*}

    \begin{tikzpicture}[overlay,remember picture]
        \draw ([yshift=2ex]pic cs:x:LASTn:transTypes:1s) to [square left brace] (pic cs:x:LASTn:transTypes:1e);
        \draw ([yshift=1ex,xshift=.1cm]pic cs:x:LASTn:transTypes:1e) to ([yshift=1ex,xshift=.4cm]pic cs:x:LASTn:transTypes:1e);
        \draw ([yshift=2ex]pic cs:x:LASTn:transTypes:2s) to [square left brace] (pic cs:x:LASTn:transTypes:2e);
        \draw ([yshift=1ex,xshift=.1cm]pic cs:x:LASTn:transTypes:2e) to ([yshift=1ex,xshift=.4cm]pic cs:x:LASTn:transTypes:2e);
        \draw ([yshift=2ex]pic cs:x:LASTn:transTypes:3s) to [square left brace] (pic cs:x:LASTn:transTypes:3e);
        \draw ([yshift=1ex,xshift=.1cm]pic cs:x:LASTn:transTypes:3e) to ([yshift=1ex,xshift=.4cm]pic cs:x:LASTn:transTypes:3e);
        \draw ([yshift=2ex]pic cs:x:LASTn:transTypes:4s) to [square left brace] (pic cs:x:LASTn:transTypes:4e);
        \draw ([yshift=1ex,xshift=.1cm]pic cs:x:LASTn:transTypes:4e) to ([yshift=1ex,xshift=.4cm]pic cs:x:LASTn:transTypes:4e);
        \draw ([yshift=2ex]pic cs:x:LASTn:transTypes:5s) to [square left brace] (pic cs:x:LASTn:transTypes:5e);
        \draw ([yshift=1ex,xshift=.1cm]pic cs:x:LASTn:transTypes:5e) to ([yshift=1ex,xshift=.4cm]pic cs:x:LASTn:transTypes:5e);
        \draw ([yshift=2ex,xshift=.2cm]pic cs:x:LASTn:transTypes:11s) to [square left brace] ([xshift=.2cm]pic cs:x:LASTn:transTypes:11e);
        \draw ([yshift=1ex,xshift=.3cm]pic cs:x:LASTn:transTypes:11e) to ([yshift=1ex,xshift=.6cm]pic cs:x:LASTn:transTypes:11e);
        \draw ([yshift=2ex,xshift=.2cm]pic cs:x:LASTn:transTypes:12s) to [square left brace] ([xshift=.2cm]pic cs:x:LASTn:transTypes:12e);
        \draw ([yshift=1ex,xshift=.3cm]pic cs:x:LASTn:transTypes:12e) to ([yshift=1ex,xshift=.6cm]pic cs:x:LASTn:transTypes:12e);
        \draw ([yshift=2ex,xshift=.2cm]pic cs:x:LASTn:transTypes:13s) to [square left brace] ([xshift=.2cm]pic cs:x:LASTn:transTypes:13e);
        \draw ([yshift=1ex,xshift=.3cm]pic cs:x:LASTn:transTypes:13e) to ([yshift=1ex,xshift=.6cm]pic cs:x:LASTn:transTypes:13e);
        \draw ([yshift=2ex,xshift=.2cm]pic cs:x:LASTn:transTypes:14s) to [square left brace] ([xshift=.2cm]pic cs:x:LASTn:transTypes:14e);
        \draw ([yshift=1ex,xshift=.3cm]pic cs:x:LASTn:transTypes:14e) to ([yshift=1ex,xshift=.6cm]pic cs:x:LASTn:transTypes:14e);
        \draw ([yshift=2ex,xshift=.2cm]pic cs:x:LASTn:transTypes:31s) to [square left brace] ([xshift=.2cm]pic cs:x:LASTn:transTypes:31e);
        \draw ([yshift=1ex,xshift=.3cm]pic cs:x:LASTn:transTypes:31e) to ([yshift=1ex,xshift=.6cm]pic cs:x:LASTn:transTypes:31e);
        \draw ([yshift=2ex,xshift=.2cm]pic cs:x:LASTn:transTypes:32s) to [square left brace] ([xshift=.2cm]pic cs:x:LASTn:transTypes:32e);
        \draw ([yshift=1ex,xshift=.3cm]pic cs:x:LASTn:transTypes:32e) to ([yshift=1ex,xshift=.6cm]pic cs:x:LASTn:transTypes:32e);
        \draw ([yshift=2ex,xshift=.4cm]pic cs:x:LASTn:transTypes:121s) to [square left brace] ([xshift=.4cm]pic cs:x:LASTn:transTypes:121e);
        \draw ([yshift=1ex,xshift=.5cm]pic cs:x:LASTn:transTypes:121e) to ([yshift=1ex,xshift=.8cm]pic cs:x:LASTn:transTypes:121e);
        \draw ([yshift=2ex,xshift=.4cm]pic cs:x:LASTn:transTypes:122s) to [square left brace] ([xshift=.4cm]pic cs:x:LASTn:transTypes:122e);
        \draw ([yshift=1ex,xshift=.5cm]pic cs:x:LASTn:transTypes:122e) to ([yshift=1ex,xshift=.8cm]pic cs:x:LASTn:transTypes:122e);
    \end{tikzpicture}
\end{example}

\paragraph{The translation: Terms, configurations, and buffers.}

Now we complete the definition of our translation.
We proceed inductively on the typing derivations of terms, configurations, and buffers.
Hence, the definition is on the typing rules in \Cref{f:LASTn:LAST:type} (including the changes discussed in \Cref{s:LASTn:LASTn:language}).
As such, typed terms and configurations are translated $\encc{z}{\normalcolor \type{\Gamma} \vdashM \term{M} : \type{T}}$ and $\encc{z}{\normalcolor \type{\Gamma} \vdashC{\phi} \term{C} : \type{T}}$, respectively, where the behavior of the term or configuration is executed on the fresh name $z$, as usual.
Typed buffers are translated $\encc{a\>b}{\normalcolor \type{\Gamma} \vdashB \term{\tBfr{\vec{m}}} : \type{S'}>\type{S}}$.
The translated buffer stores the translations of the messages in~$\term{\vec{m}}$ and forwards them on fresh name $b$; the  buffer receives new messages on the fresh name~$a$.
Note that, if the buffer is empty, the roles of $a$ and $b$ may be reversed (depending on $\type{S}$), and the buffer may receive new messages on $b$ instead.

Importantly, the translation is \emph{type preserving} using the translations of types introduced in \Cref{d:LASTn:LASTn:transTypes}.
\ifappendix\Cref{as:LASTn:transTp} \else\cite{thesis/vdHeuvel24ex} \fi contains a detailed proof.

\begin{restatable}[Type Preservation for the Translation]{theorem}{tLASTnTransTp}
\label{t:LASTn:LASTn:transTp}
    \leavevmode
    \begin{itemize}

        \item
            If\;\;$\type{\Gamma} \vdashM \term{M} : \type{T}$, then $\vdash* \encc{z}{\term{M}} \typInf \enct*{\Gamma} , z:\enct{T}$;

        \item
            If\;\;$\type{\Gamma} \vdashC{\phi} \term{C} : \type{T}$, then $\vdash* \encc{z}{\term{C}} \typInf \enct*{\Gamma} , z:\enct{T}$;

        \item
            If\;\;$\type{\Gamma} \vdashB \term{\tBfr{\vec{m}}} : \type{S'}>\type{S}$, then $\vdash* \encc{a\>b}{\tBfr{\vec{m}}} \typInf \enct*{\Gamma} , a:\ol{\enct*{S'}} , b:\ol{\enct*{\ol{S}}}$.

    \end{itemize}
\end{restatable}

\begin{figure}[p]
    \vspace{-7pt}

    \begin{align*}
        \arrayrulecolor{gray!50}
        \setlength{\arrayrulewidth}{.1pt}
        & \tinyRL{typ-var}
        &
        \encc{z}{x}
        &\deq
        \pOut x[\_,z]
        \\[1pt] \hline \\[-14pt]
        & \tinyRL{typ-abs}
        &
        \encc{z}{\lam x . M}
        &\deq
        \pIn z(x,a) ; \encc{a}{M}
        \tag*{receive $x$, then run body}
        \\[1pt] \hline \\[-14pt]
        & \tinyRL{typ-app}
        &
        \encc{z}{M~N}
        &\deq
        \pRes{ab}\pRes{cd} (
        \encc{a}{M}
        \tag*{run abstraction}
        \\ &&&\hphantom{{}\deq{}}
        ~~ \| \pOut b[c,z]
        \tag*{trigger function body}
        \\ &&&\hphantom{{}\deq{}}
        ~~ \| \pIn d(\_,e) ; \encc{e}{N} )
        \tag*{parameter as future substitution}
        \\[1pt] \hline \\[-14pt]
        & \tinyRL{typ-unit}
        &
        \encc{z}{()}
        &\deq \0
        \\[1pt] \hline \\[-14pt]
        & \tinyRL{typ-pair}
        &
        \encc{z}{(M,N)}
        &\deq
        \pRes{ab} \pRes{cd} (
        \pOut z[a,c]
        \tag*{announce pair is ready}
        \\ &&&\hphantom{{}\deq{}}
        ~~ \| \pIn b(\_,e) ; \encc{e}{M} \| \pIn d(\_,f) ; \encc{f}{N} )
        \tag*{components as future substitutions}
        \\[1pt] \hline \\[-14pt]
        & \tinyRL{typ-split}
        &
        \encc{z}{\tLet (x,y) = M \tIn N}
        &\deq
        \pRes{ab} (
        \pIn a(x,y) ; \encc{z}{N}
        \tag*{block body until pair ready}
        \\ &&&\hphantom{{}\deq{}}
        ~~ \| \encc{b}{M} )
        \tag*{run pair}
        \\[1pt] \hline \\[-14pt]
        & \tinyRL{typ-new}
        &
        \encc{z}{\tNew}
        &\deq
        \pRes{ab} (
        \pOut a[\_,z]
        \tag*{activate buffer}
        \\ &&&\hphantom{{}\deq{}}
        ~~ \| \pIn b(\_,c) ; \pRes{dx} \pRes{ey} (
        \tag*{block until activated}
        \\ &&&\hphantom{{}\deq{}}
        ~~~~ \encc{d\>e}{\tBfr{\epsi}}
        \tag*{prepare buffer}
        \\ &&&\hphantom{{}\deq{}}
        ~~~~ \| \encc{c}{(x,y)} ) )
        \tag*{return pair of endpoints}
        \\[1pt] \hline \\[-14pt]
        & \tinyRL{typ-spawn}
        &
        \encc{z}{\tSpawn M ; N}
        &\deq
        \pRes{ab} (
        \pOut a[\_,z]
        \tag*{activate bodies}
        \\ &&&\hphantom{{}\deq{}}
        ~~ \| \pIn b(\_,c) ; (
        \tag*{block until activated}
        \\ &&&\hphantom{{}\deq{}}
        ~~~~ \pRes{\_\_} \encc{\_}{M} \| \encc{c}{N} ) )
        \tag*{run bodies}
        \\[1pt] \hline \\[-14pt]
        & \tinyRL{typ-close}
        &
        \encc{z}{\tClose M ; N}
        &\deq
        \pRes{ab} (
        \encc{a}{M}
        \tag*{run argument to activate buffer}
        \\ &&&\hphantom{{}\deq{}}
        ~~ \| \pIn b(\_,\_) ; \encc{z}{N} )
        \tag*{wait for buffer to close}
        \\[1pt] \hline \\[-14pt]
        & \tinyRL{typ-send}
        &
        \encc{z}{\tSend M \, N}
        &\deq
        \pRes{ab} \pRes{cd} (
        \pIn a(\_,e) ; \encc{e}{M}
        \tag*{block payload until received}
        \\ &&&\hphantom{{}\deq{}}
        ~~ \| \encc{c}{N}
        \tag*{run channel term to activate buffer}
        \\ &&&\hphantom{{}\deq{}}
        ~~ \| \pIn d(\_,f) ; \pRes{gh} (
        \tag*{wait for buffer to activate}
        \\ &&&\hphantom{{}\deq{}}
        ~~~~ \pOut f[b,g]
        \tag*{send to buffer}
        \\ &&&\hphantom{{}\deq{}}
        ~~~~ \| \pOut h[\_,z] ) )
        \tag*{prepare returned endpoint variable}
        \\[1pt] \hline \\[-14pt]
        & \tinyRL{typ-recv}
        &
        \encc{z}{\tRecv M}
        &\deq
        \pRes{ab} (
        \encc{a}{M}
        \tag*{run channel term to activate buffer}
        \\ &&&\hphantom{{}\deq{}}
        ~~ \| \pIn b(c,d) ;
        \tag*{receive from buffer}
        \\ &&&\hphantom{{}\deq{}}
        ~~~~ \pRes{ef} ( \pOut z[c,e] \| \pIn f(\_,g) ; d[\_,g] ) )
        \tag*{returned pair}
        \\[1pt] \hline \\[-14pt]
        & \tinyRL{typ-sel}
        &
        \encc{z}{\tSel j \, M}
        &\deq
        \pRes{ab} (
        \encc{a}{M}
        \tag*{run channel term to activate buffer}
        \\ &&&\hphantom{{}\deq{}}
        ~~ \| \pIn b(\_,c) ; \pRes{de} (
        \tag*{wait for buffer to activate}
        \\ &&&\hphantom{{}\deq{}}
        ~~~~ \pSel c[d] < j
        \tag*{select with buffer}
        \\ &&&\hphantom{{}\deq{}}
        ~~~~ \| e[\_,z] ) )
        \tag*{prepare returned endpoint variable}
    \end{align*}
    
    \vspace{-7pt}

    \caption{Translating \protect\LASTn into \APCP, Part 1/3.}
    \label{f:LASTn:LASTn:trans1}
\end{figure}

\begin{figure}[p]
    \begin{align*}
        \arrayrulecolor{gray!50}
        \setlength{\arrayrulewidth}{.1pt}
        & \tinyRL{typ-case}
        &
        \encc{z}{\tCase M \tOf \{ i : N_i \}_{i \in I}}
        &\deq
        \pRes{ab} (
        \encc{a}{M}
        \tag*{run channel term to activate buffer}
        \\ &&&\hphantom{{}\deq{}}
        ~~ \| \pBra b(c) > \{ i :
        \tag*{branch on buffer}
        \\ &&&\hphantom{{}\deq{}}
        ~~~~ \encc{z}{N_i\ c} \}_{i \in I} )
        \tag*{apply continuation to endpoint}
        \\[1pt] \hline \\[-14pt]
        & \tinyRL{typ-sub}
        &
        \encc{z}{M \tSub{ N/x }}
        &\deq
        \pRes{xa} (
        \encc{z}{M}
        \tag*{run body}
        \\ &&&\hphantom{{}\deq{}}
        ~~ \pIn a(\_,b) ; \encc{b}{N} )
        \tag*{block until body is variable}
        \\[1pt] \hline \\[-14pt]
        & \tinyRL{typ-main}
        &
        \encc{z}{\tMain \, M}
        &\deq
        \encc{z}{M}
        \\[1pt] \hline \\[-14pt]
        & \tinyRL{typ-child}
        &
        \encc{z}{\tChild \, M}
        &\deq
        \encc{z}{M}
        \\[1pt] \hline \\[-14pt]
        & \tinyRL{typ-par}~(\type{T_1} = \type{\1})
        &
        \encc{z}{C \prl D}
        &\deq
        \pRes{\_\_} \encc{\_}{C} \| \encc{z}{D}
        \\[1pt] \hline \\[-14pt]
        & \tinyRL{typ-par}~(\type{T_2} = \type{\1})
        &
        \encc{z}{C \prl D}
        &\deq
        \encc{z}{C} \| \pRes{\_\_} \encc{\_}{D}
        \\[1pt] \hline \\[-14pt]
        & \tinyRL{typ-res}
        &
        \encc{z}{\pRes{x\tBfr{\vec{m}}y} C}
        &\deq
        \pRes{ax} \pRes{by} (
        \encc{a\>b}{\tBfr{\vec m}}
        \tag*{run buffer}
        \\ &&&\hphantom{{}\deq{}}
        ~~ \| \encc{z}{C} )
        \tag*{run configuration}
        \\[1pt] \hline \\[-14pt]
        & \tinyRL{typ-conf-sub}
        &
        \encc{z}{C \tSub{ N/x }}
        &\deq
        \pRes{xa} (
        \encc{z}{C}
        \tag*{run configuration}
        \\ &&&\hphantom{{}\deq{}}
        ~~ \| \pIn a(\_,b) ; \encc{b}{N}
        \tag*{block until configuration is variable}
        \\[1pt] \hline \\[-14pt]
        & \tinyRL{typ-buf}~(\type{S'} = \type{{!}T.S})
        &
        \encc{a\>b}{\tBfr{\epsi}}
        &\deq
        \pIn a(\_,c) ; \pRes{de} (
        \tag*{wait for activation}
        \\ &&&\hphantom{{}\deq{}}
        ~~ \pOut c[\_,d]
        \tag*{activate send}
        \\ &&&\hphantom{{}\deq{}}
        ~~ \| \pIn e(f,g) ; \pRes{hk} (
        \tag*{receive from send}
        \\ &&&\hphantom{{}\deq{}}
        ~~~~ \pIn b(\_,l) ;
        \tag*{wait for activation}
        \\ &&&\hphantom{{}\deq{}}
        ~~~~~~ \pOut l[f,h]
        \tag*{send to receive}
        \\ &&&\hphantom{{}\deq{}}
        ~~~~ \| \enc{g\>k}{\term{\tBfr{\epsi}} : \type{S}>\type{S}} ) )
        \\[1pt] \hline \\[-14pt]
        & \tinyRL{typ-buf}~(\type{S'} = \type{\oplus \{ i : S_i \}_{i \in I}})
        &
        \encc{a\>b}{\tBfr{\epsi}}
        &\deq
        \pIn a(\_,c) ; \pRes{de} (
        \tag*{wait for activation}
        \\ &&&\hphantom{{}\deq{}}
        ~~ \pOut c[\_,d]
        \tag*{activate select}
        \\ &&&\hphantom{{}\deq{}}
        ~~ \| \pBra e(f) > \{ i : \pRes{gh} (
        \tag*{receive selection}
        \\ &&&\hphantom{{}\deq{}}
        ~~~~ \pIn b(\_,k) ;
        \tag*{wait for activation}
        \\ &&&\hphantom{{}\deq{}}
        ~~~~~~ \pSel k[g] < i
        \tag*{make selection}
        \\ &&&\hphantom{{}\deq{}}
        ~~~~ \| \enc{f\>h}{\term{\tBfr{\epsi}} : \type{S_i}>\type{S_i}} ) \} )
        \\[1pt] \hline \\[-14pt]
        & \begin{array}[t]{@{}l@{}}
            \tinyRL{typ-buf}
            \\
            (\type{S'} \in \{\type{{?}T.S},\type{\& \{ i : S_i \}_{i \in I}}\})
        \end{array}
        &
        \encc{a\>b}{\tBfr{\epsi}}
        &\deq
        \enc{b\>a}{\term{\tBfr{\epsi}} : \type{\ol{S'}}>\type{\ol{S'}}}
    \end{align*}

    \caption{Translating \protect\LASTn into \APCP, Part 2/3.}
    \label{f:LASTn:LASTn:trans2}
\end{figure}

\begin{figure}[p]
    \begin{align*}
        & \tinyRL{typ-buf}~(\type{S'} = \type{\tEnd})
        &
        \encc{a\>b}{\tBfr{\epsi}}
        &\deq
        \pIn a(\_,c) ; \pOut c[\_,\_] \| \pIn b(\_,d) ; \pOut d[\_,\_]
        \tag*{close concurrently}
        \\[1pt] \hline \\[-14pt]
        & \tinyRL{typ-buf}~(\type{S'} = \type{\tNil})
        &
        \encc{a\>b}{\tBfr{\epsi}}
        &\deq
        \0
        \\[1pt] \hline \\[-14pt]
        & \tinyRL{typ-buf-send}
        &
        \encc{a\>b}{\tBfr{\vec{m},M}}
        &\deq
        \pRes{cd} \pRes{ef} (
        \pIn c(\_,g) ; \encc{g}{M}
        \tag*{block payload until received}
        \\ &&&\hphantom{{}\deq{}}
        ~~ \| \pIn b(\_,h) ;
        \tag*{wait for activation}
        \\ &&&\hphantom{{}\deq{}}
        ~~~~ \pOut h[d,e]
        \tag*{send to receive}
        \\ &&&\hphantom{{}\deq{}}
        ~~ \| \encc{a\>f}{\tBfr{\vec{m}}} )
        \tag*{run buffer}
        \\[1pt] \hline \\[-14pt]
        & \tinyRL{typ-buf-sel}
        &
        \encc{a\>b}{\tBfr{\vec{m},j}}
        &\deq
        \pRes{cd} (
        \pIn b(\_,e) ;
        \tag*{wait for activation}
        \\ &&&\hphantom{{}\deq{}}
        ~~~~ \pSel e[c] < j
        \tag*{make selection}
        \\ &&&\hphantom{{}\deq{}}
        ~~ \| \encc{a\>d}{\tBfr{\vec{m}}} )
        \tag*{run buffer}
        \\[1pt] \hline \\[-14pt]
        & \tinyRL{typ-buf-end-L}
        &
        \encc{a\>b}{\tBfr{\epsi}}
        &\deq
        \pIn a(\_,c) ; \pOut c[\_,\_]
        \\[1pt] \hline \\[-14pt]
        & \tinyRL{typ-buf-end-R}
        &
        \encc{a\>b}{\tBfr{\epsi}}
        &\deq
        \pIn b(\_,c) ; \pOut c[\_,\_]
    \end{align*}

    \caption{Translating \protect\LASTn into \APCP, Part 3/3.}
    \label{f:LASTn:LASTn:trans3}
\end{figure}

\Cref{f:LASTn:LASTn:trans1,f:LASTn:LASTn:trans2,f:LASTn:LASTn:trans3} give the translation of term, configuration, and buffer typing rules.
The figures omit types; they follow \Cref{d:LASTn:LASTn:transTypes} and \Cref{t:LASTn:LASTn:transTp}.
The figures include line-by-line explanations of the translation.
Recall that we often write~`$\_$' to denote a fresh name of type $\bullet$; when sending names denoted `$\_$', we omit binders~`$\pRes{\_\_}$'.

\begin{figure}[p]
    \begin{align*}
        & \encc{z}{
            \lam s_2 . \tLet (\textit{card},s_3) = \tRecv s_2 \tIn 
            \tLet s_4 = \tSend \text{book}(\textit{title}) \, s_3 \tIn 
            \tClose s_4 ; ()
        }
        \span\span
        \\
        {}={}& \pIn z(s_2,a_0) ;
        &&
        \tikzmark{x:LASTn:trans:1s}
        \tag*{$\encc{z}{\lam s_2 \ldots}$}
        \\
        & \q[1] \pRes{a_1b_1} ( \pIn a_1(\textit{card},s_3) ;
        &&
        \tikzmark{x:LASTn:trans:11s}
        \tag*{$\encc{a_0}{\tLet (\textit{card},s_3) = \ldots}$}
        \\
        & \q[2] \pRes{a_2b_2} \pRes{c_2d_2} (
        &&
        \tikzmark{x:LASTn:trans:111s}
        \tag*{$\encc{a_0}{\tLet s_4 = \ldots}$}
        \\
        & \q[3] \pIn a_2(s_4,a_3) ;
        &&
        \tikzmark{x:LASTn:trans:1111s}
        \tag*{$\encc{a_2}{\lam s_4 . \ldots}$}
        \\
        & \q[4] \pRes{a_4b_4} (
        &&
        \tikzmark{x:LASTn:trans:11111s}
        \tag*{$\encc{a_3}{\tClose \ldots}$}
        \\
        & \q[5] \pOut s_4[\_,a_4]
        &&
        \tikzmark{x:LASTn:trans:111111s}
        \tikzmark{x:LASTn:trans:111111e}
        \tag*{$\encc{a_4}{s_4}$}
        \\
        & \q[5] \| \pIn b_4(\_,\_) ;
        \\
        & \q[6] \0 )
        &&
        \tikzmark{x:LASTn:trans:111112s}
        \tikzmark{x:LASTn:trans:111112e}
        \tikzmark{x:LASTn:trans:11111e}
        \tikzmark{x:LASTn:trans:1111e}
        \tag*{$\encc{a_3}{()}$}
        \\
        & \q[3] \pOut b_2[c_2,a_0] \| \pIn d_2(\_,e_2) ;
        \\
        & \q[4] \pRes{a_5b_5} \pRes{c_5d_5} ( \pIn a_5(\_,e_5) ;
        &&
        \tikzmark{x:LASTn:trans:1112s}
        \tag*{$\encc{e_2}{\tSend \ldots}$}
        \\
        & \q[5] \encc{e_5}{\text{book}(\textit{title})}
        \\
        & \q[5] \| \pOut s_3[\_,c_5]
        &&
        \tikzmark{x:LASTn:trans:11121s}
        \tikzmark{x:LASTn:trans:11121e}
        \tag*{$\encc{c_5}{s_3}$}
        \\
        & \q[5] \| \pIn d_5(\_,f_5) ; \pRes{g_5h_5} ( \pOut f_5[b_5,g_5] \| \pOut h_5[\_,e_2] ) ) )
        &&
        \tikzmark{x:LASTn:trans:1112e}
        \tikzmark{x:LASTn:trans:111e}
        \\
        & \q[2] \| \pRes{a_6b_6} (
        &&
        \tikzmark{x:LASTn:trans:112s}
        \tag*{$\encc{b_1}{\tRecv \ldots}$}
        \\
        & \q[3] \pOut s_2[\_,a_6]
        &&
        \tikzmark{x:LASTn:trans:1121s}
        \tikzmark{x:LASTn:trans:1121e}
        \tag*{$\encc{a_6}{s_2}$}
        \\
        & \q[3] \| \pIn b_6(c_6,d_6) ; \pRes{e_6f_6} ( \pOut b_1[c_6,e_6] \| \pIn f_6(\_,g_6) ; \pOut d_6[\_,g_6] ) ) )
        \hspace{-1.90em}
        && 
        \tikzmark{x:LASTn:trans:1e}
        \tikzmark{x:LASTn:trans:112e}
        \tikzmark{x:LASTn:trans:11e}
        \q[15]
    \end{align*}%
    
    \begin{tikzpicture}[overlay,remember picture]
        \draw ([yshift=2ex]pic cs:x:LASTn:trans:1s) to [square left brace] (pic cs:x:LASTn:trans:1e);
        \draw ([yshift=1ex,xshift=.1cm]pic cs:x:LASTn:trans:1s) to ([yshift=1ex,xshift=.4cm]pic cs:x:LASTn:trans:1s);
        \draw ([yshift=2ex,xshift=.2cm]pic cs:x:LASTn:trans:11s) to [square left brace] ([xshift=.2cm]pic cs:x:LASTn:trans:11e);
        \draw ([yshift=1ex,xshift=.3cm]pic cs:x:LASTn:trans:11s) to ([yshift=1ex,xshift=.6cm]pic cs:x:LASTn:trans:11s);
        \draw ([yshift=2ex,xshift=.4cm]pic cs:x:LASTn:trans:111s) to [square left brace] ([xshift=.4cm]pic cs:x:LASTn:trans:111e);
        \draw ([yshift=1ex,xshift=.5cm]pic cs:x:LASTn:trans:111s) to ([yshift=1ex,xshift=.8cm]pic cs:x:LASTn:trans:111s);
        \draw ([yshift=2ex,xshift=.4cm]pic cs:x:LASTn:trans:112s) to [square left brace] ([xshift=.4cm]pic cs:x:LASTn:trans:112e);
        \draw ([yshift=1ex,xshift=.5cm]pic cs:x:LASTn:trans:112s) to ([yshift=1ex,xshift=.8cm]pic cs:x:LASTn:trans:112s);
        \draw ([yshift=2ex,xshift=.6cm]pic cs:x:LASTn:trans:1111s) to [square left brace] ([xshift=.6cm]pic cs:x:LASTn:trans:1111e);
        \draw ([yshift=1ex,xshift=.7cm]pic cs:x:LASTn:trans:1111s) to ([yshift=1ex,xshift=1.0cm]pic cs:x:LASTn:trans:1111s);
        \draw ([yshift=2ex,xshift=.6cm]pic cs:x:LASTn:trans:1112s) to [square left brace] ([xshift=.6cm]pic cs:x:LASTn:trans:1112e);
        \draw ([yshift=1ex,xshift=.7cm]pic cs:x:LASTn:trans:1112s) to ([yshift=1ex,xshift=1.0cm]pic cs:x:LASTn:trans:1112s);
        \draw ([yshift=2ex,xshift=.6cm]pic cs:x:LASTn:trans:1121s) to [square left brace] ([xshift=.6cm]pic cs:x:LASTn:trans:1121e);
        \draw ([yshift=1ex,xshift=.7cm]pic cs:x:LASTn:trans:1121s) to ([yshift=1ex,xshift=1.0cm]pic cs:x:LASTn:trans:1121s);
        \draw ([yshift=2ex,xshift=.8cm]pic cs:x:LASTn:trans:11111s) to [square left brace] ([xshift=.8cm]pic cs:x:LASTn:trans:11111e);
        \draw ([yshift=1ex,xshift=.9cm]pic cs:x:LASTn:trans:11111s) to ([yshift=1ex,xshift=1.2cm]pic cs:x:LASTn:trans:11111s);
        \draw ([yshift=2ex,xshift=.8cm]pic cs:x:LASTn:trans:11121s) to [square left brace] ([xshift=.8cm]pic cs:x:LASTn:trans:11121e);
        \draw ([yshift=1ex,xshift=.9cm]pic cs:x:LASTn:trans:11121s) to ([yshift=1ex,xshift=1.2cm]pic cs:x:LASTn:trans:11121s);
        \draw ([yshift=2ex,xshift=1.0cm]pic cs:x:LASTn:trans:111111s) to [square left brace] ([xshift=1.0cm]pic cs:x:LASTn:trans:111111e);
        \draw ([yshift=1ex,xshift=1.1cm]pic cs:x:LASTn:trans:111111s) to ([yshift=1ex,xshift=1.4cm]pic cs:x:LASTn:trans:111111s);
        \draw ([yshift=2ex,xshift=1.0cm]pic cs:x:LASTn:trans:111112s) to [square left brace] ([xshift=1.0cm]pic cs:x:LASTn:trans:111112e);
        \draw ([yshift=1ex,xshift=1.1cm]pic cs:x:LASTn:trans:111112s) to ([yshift=1ex,xshift=1.4cm]pic cs:x:LASTn:trans:111112s);
    \end{tikzpicture}
    
    \caption{Partial translation of the bookshop (cf.\ \Cref{x:LASTn:LASTn:trans}).}\label{f:LASTn:LASTn:trans}
\end{figure}

\begin{example}
    \label{x:LASTn:LASTn:trans}
    \Cref{f:LASTn:LASTn:trans} illustratse our translation by translating a subterm of the shop defined in \Cref{s:LASTn:example} (leaving the translation of the primitive book-lookup $\term{\text{book}(\textit{title})}$ unspecified).
    Below, we translate the outer buffer in the final state of the system in \Cref{x:LASTn:LASTn:Shop}:
    \begin{align*}
        & \encc{z}{\pRes{y'\tBfr{\text{visa},\sff{buy},\text{``Dune''}}y} \big( \pRes{z'\tBfr{\epsi}z} \ldots \prl \tChild \, \sff{Shop}_y \big)}
        \span\span
        \\
        {}={}
        & \pRes{a_1y'} \pRes{b_1y} (
        &&
        \tikzmark{x:LASTn:trans2:1s}
        \tag*{$\encc{z}{\pRes{y'\tBfr{\ldots}y} \ldots}$}
        \\
        & \q[1] \pRes{c_1d_1} \pRes{e_1f_1} ( \pIn c_1(\_,g_1) ;
        &&
        \tikzmark{x:LASTn:trans2:11s}
        \tag*{$\encc{a_1\>b_1}{\tBfr{\ldots,\text{``Dune''}}}$}
        \\
        & \q[2] \encc{g_1}{\text{``Dune''}}
        \\
        & \q[2] {} \| \pIn b_1(\_,h_1) ; \pOut h_1[d_1,e_1]
        \\
        & \q[2] {} \| \pRes{c_2d_2} ( \pIn f_1(\_,e_2) ; \pSel e_2[c_2] < \sff{buy}
        &&
        \tikzmark{x:LASTn:trans2:111s}
        \tag*{$\encc{a_1\>f_1}{\tBfr{\ldots,\sff{buy}}}$}
        \\
        & \q[3] {} \| \pRes{c_3d_3} \pRes{e_3f_3} ( \pIn c_3(\_,g_3) ;
        &&
        \tikzmark{x:LASTn:trans2:1111s}
        \tag*{$\encc{a_1\>d_2}{\tBfr{\text{visa}}}$}
        \\
        & \q[4] \encc{g_3}{\text{visa}}
        \\
        & \q[4] {} \| \pIn d_2(\_,h_3) ; \pOut h_3[d_3,e_3]
        \\
        & \q[4] {} \| \pIn f_3(\_,c_4) ; \pRes{d_4e_4} ( \pOut c_4[\_,d_4] \| \pIn e_4(f_4,g_4) ; 
        &&
        \tikzmark{x:LASTn:trans2:11111s}
        \tag*{$\enc{f_3\>a_1}{\term{\tBfr{\epsi}} : \type{{!}\sff{B}.\tEnd} > \type{{!}\sff{B}.\tEnd}}$}
        \\
        & \q[5] \pRes{h_4k_4} ( \pIn a_1(\_,l_4) ; \pOut l_4[f_4,h_4]     
        \\
        & \q[6] {} \| \pIn g_4(\_,c_5) ; \pOut c_5[\_,\_] \| \pIn k_4(\_,d_5) ; \pOut d_5[\_,\_] ) ) ) ) )
        \hspace{-4em}
        &&
        \tikzmark{x:LASTn:trans2:111111s}
        \tikzmark{x:LASTn:trans2:111111e}
        \tikzmark{x:LASTn:trans2:11111e}
        \tikzmark{x:LASTn:trans2:1111e}
        \tikzmark{x:LASTn:trans2:111e}
        \tikzmark{x:LASTn:trans2:11e}
        \hspace{5em}
        \tag*{$\enc{g_4\>k_4}{\term{\tBfr{\epsi}} : \type{\tEnd} > \type{\tEnd}}$}
        \\
        & \q[1] {} \| \encc{z}{\pRes{z'\tBfr{\epsi}z} \ldots \prl \tChild \, \sff{Shop}_y}
        &&
        \tikzmark{x:LASTn:trans2:1e}
        \tag*{\qedhere}
    \end{align*}

    \begin{tikzpicture}[overlay,remember picture]
        \draw ([yshift=2ex]pic cs:x:LASTn:trans2:1s) to [square left brace] (pic cs:x:LASTn:trans2:1e);
        \draw ([yshift=1ex,xshift=.1cm]pic cs:x:LASTn:trans2:1s) to ([yshift=1ex,xshift=.4cm]pic cs:x:LASTn:trans2:1s);
        \draw ([yshift=2ex,xshift=.2cm]pic cs:x:LASTn:trans2:11s) to [square left brace] ([xshift=.2cm]pic cs:x:LASTn:trans2:11e);
        \draw ([yshift=1ex,xshift=.3cm]pic cs:x:LASTn:trans2:11s) to ([yshift=1ex,xshift=.6cm]pic cs:x:LASTn:trans2:11s);
        \draw ([yshift=2ex,xshift=.4cm]pic cs:x:LASTn:trans2:111s) to [square left brace] ([xshift=.4cm]pic cs:x:LASTn:trans2:111e);
        \draw ([yshift=1ex,xshift=.5cm]pic cs:x:LASTn:trans2:111s) to ([yshift=1ex,xshift=.8cm]pic cs:x:LASTn:trans2:111s);
        \draw ([yshift=2ex,xshift=.6cm]pic cs:x:LASTn:trans2:1111s) to [square left brace] ([xshift=.6cm]pic cs:x:LASTn:trans2:1111e);
        \draw ([yshift=1ex,xshift=.7cm]pic cs:x:LASTn:trans2:1111s) to ([yshift=1ex,xshift=1.0cm]pic cs:x:LASTn:trans2:1111s);
        \draw ([yshift=2ex,xshift=.8cm]pic cs:x:LASTn:trans2:11111s) to [square left brace] ([xshift=.8cm]pic cs:x:LASTn:trans2:11111e);
        \draw ([yshift=1ex,xshift=.9cm]pic cs:x:LASTn:trans2:11111s) to ([yshift=1ex,xshift=1.2cm]pic cs:x:LASTn:trans2:11111s);
        \draw ([yshift=2ex,xshift=1.0cm]pic cs:x:LASTn:trans2:111111s) to [square left brace] ([xshift=1.0cm]pic cs:x:LASTn:trans2:111111e);
        \draw ([yshift=1ex,xshift=1.1cm]pic cs:x:LASTn:trans2:111111s) to ([yshift=1ex,xshift=1.4cm]pic cs:x:LASTn:trans2:111111s);
    \end{tikzpicture}
\end{example}

\paragraph{Design decisions: Explicit substitutions and closing.}
Having presented our typed translation, we reflect on our decision to enrich \LASTn with explicit substitutions and closing.

As we will see next, adopting explicit substitutions leads to a direct operational correctness result.
To see why, suppose we were to apply substitutions immediately.
Then, e.g., the translation of function application (\Cref{f:LASTn:LASTn:trans1}) needs to somehow encode that substitution of a variable in the translation of a function for the translation of the parameter.
However, if this variable does not occur under an evaluation context, there is no way in \APCP to perform such a substitution immediately.
Hence, the translation would still need to encode such \emph{implicit} substitutions in \LASTn \emph{explicitly} in \APCP.
This discrepancy would then have to be handled by means of the aforementioned substitution lifting in our operational correspondence results.
Although this is a perfectly viable approach, we prefer to have a more direct operational correspondence, which entails more direct proofs that are not affected by an asymmetric treatment of substitutions.

Our choice for semi-explicit closing of sessions is more pragmatic: it leads to a compact translation.
Suppose we were to treat closed sessions by silently weakening them.
Consider, e.g., Rule~\ruleLabel{red-send} (\Cref{f:LASTn:LAST:confs}): the $\term{\tSend}$-primitive is replaced by a variable pointing to the buffered channel, even if the session ends after the send.
Hence, the translation would need a separate case for translating variables for closed sessions.
Consequently, the translation would need similar such separate cases anywhere variables/closed sessions may occur.
This would lead to a translation where the key ideas are unnecessarily obfuscated.
It is clearly preferable to treat closed sessions semi-explicitly in favor of a more streamlined translation.

\paragraph{Faithfulness: Operational correspondence.}
Following the discussion in \Cref{s:LASTn:LAST:trans}, here we finally show that the translation presented above satisfies the operational correctness criterion (completeness and soundness) as proposed by Gorla~\cite{journal/ic/Gorla10}.
We first state both results, and then their proofs.

\begin{restatable}[Completeness]{theorem}{tLASTnTransCompl}\label{t:LASTn:LASTn:transCompl}
    Given $\type{\Gamma} \vdashC{\phi} \term{C}: \type{T}$, if $\term{C} \reddC \term{D}$, then $\encc{z}{C} \redd^\ast \encc{z}{D}$.
\end{restatable}

\vspace{-2em}

\begin{restatable}[Soundness]{theorem}{tLASTnTransSound}
\label{t:LASTn:LASTn:transSound}
    Given $\type{\Gamma} \vdashC{\phi} \term{C} : \type{T}$, if $\encc{z}{C} \redd* Q$, then there exists $\term{D}$ such that $\term{C} \reddC^\ast \term{D}$ and $Q \redd* \encc{z}{D}$.
\end{restatable}

Both results rely on the following lemma, which decomposes translations of terms/configurations under contexts into some evaluation context containing the translation of the respective term/configuration, and transfers free variables/names and substitutions in and out of the translation\ifappendix (cf.\ \Cref{as:LASTn:oc})\fi:

\begin{restatable}{lemma}{lLASTnTransCtxs}
\label{l:LASTn:LASTn:transCtxs}
    \leavevmode
    \begin{itemize}

        \item
            $\encc{z}{\tCtx{R}[M]} = \evalCtx{E}[\encc{z'}{M}]$ for some $\evalCtx{E},z'$;

        \item
            $\encc{z}{\tCtx{F}[M]} = \evalCtx{E}[\encc{z'}{M}]$ for some $\evalCtx{E},z'$;

        \item
            $\encc{z}{\tCtx{G}[C]} = \evalCtx{E}[\encc{z'}{C}]$ for some $\evalCtx{E},z'$;
            
        \item 
            $\term{x} \notin \fv(\term{C})$ implies $x \notin \fn(\encc{z}{C})$.
            
        \item 
            $\term{x} \in \fv(\term{C})$ implies $\encc{z}{C \{ y/x \}} = \encc{z}{C} \{ y/x \}$.

    \end{itemize}
\end{restatable}

Completeness additionally relies on the following lemma, that decomposes the translation of buffers of several shapes into an evaluation context containing the translation of a continuation of the respective buffer\ifappendix (cf.~\Cref{as:LASTn:oc:compl})\fi:

\begin{restatable}{lemma}{lLASTnTransBuf}
    \label{l:LASTn:LASTn:transBuf}
    \leavevmode
    \begin{itemize}

        \item
            $\type{S'} \neq \type{\tNil}$ implies $\encc{a\>b}{\type{\Gamma} \vdashB \term{\tBfr{\vec{m}}} : \type{S'} > \type{S}} = \evalCtx[\big]{E}[\encc{a\>c}{\type{\emptyset} \vdashB \term{\tBfr{\epsi}} : \type{S'}>\type{S'}}]$

        \item
            $\type{S'} \neq \type{\tNil}$ implies $\encc{a\>b}{\type{\Gamma},\type{\Delta} \vdashB \term{\tBfr{M,\vec{m}}} : \type{S'} > \type{S}} = \evalCtx[\big]{E}[\encc{a\>c}{\type{\Delta} \vdashB \term{\tBfr{M}} : \type{S'}>\type{{!}T.S'}}]$

        \item
            $\type{S'} \neq \type{\tNil}$ implies $\encc{a\>b}{\type{\Gamma} \vdashB \term{\tBfr{j,\vec{m}}} : \type{S'} > \type{S}} = \evalCtx[\big]{E}[\encc{a\>c}{\type{\emptyset} \vdashB \term{\tBfr{j}} : \type{S'}>\type{\oplus \{ i : S_i \}_{i \in I} \cup \{ j : S' \}}}]$

        \item
            $\type{S} \neq \type{\tNil}$ implies $\encc{a\>b}{\type{\Gamma} \vdashB \term{\tBfr{\vec{m}}} : \type{\tEnd} > \type{S}} = \evalCtx[\big]{E}[\encc{a\>c}{\type{\emptyset} \vdashB \term{\tBfr{\epsi}} : \type{\tEnd}>\type{\tEnd}}]$

        \item
            $\type{S} \neq \type{\tNil}$ implies $\encc{a\>b}{\type{\Gamma} \vdashB \term{\tBfr{\vec{m}}} : \type{\tNil} > \type{S}} = \evalCtx[\big]{E}[\encc{a\>c}{\type{\emptyset} \vdashB \term{\tBfr{\epsi}} : \type{\tNil}>\type{\tEnd}}]$

    \end{itemize}
    for some $\evalCtx{E},c$.
\end{restatable}

With these pre-requisites in place, we discuss the proofs of \Cref{t:LASTn:LASTn:transCompl,t:LASTn:LASTn:transSound}. 

\begin{proof}[Proof of completeness (\Cref{t:LASTn:LASTn:transCompl})]
    By induction on the derivation of $\term{C} \reddC \term{D}$.
    We discuss one representative rule for communication, as well as the structural rules\ifappendix; other rules are detailed in \Cref{as:LASTn:oc:compl}\fi.
    \begin{itemize}

        \item 
            \global\def\proofTransComplSend{%
                Rule~\ruleLabel{red-send}:
                $\term{\pRes{x\tBfr{\vec{m}}y} ( \tCtx{\hat{F}}[\tSend M \, x] \prl C )} \reddC \term{\pRes{x\tBfr{M,\vec{m}}y} ( \tCtx{\hat{F}}[x] \prl C )}$.
                W.l.o.g., assume $\term{C}$ is a child thread.
                By \Cref{l:LASTn:LASTn:transCtxs}, for any $\term{L}$, $\encc{z}{\tCtx{F}[L]} = \evalCtx[\big]{E_1}[\encc{z'}{L}]$ for some $\evalCtx{E_1},z'$ ($\ast_1$).
                Moreover, since $\tCtx{\hat{F}}$ does not have its hole under an explicit substitution, it does not capture any free variables of $M$; hence, by \Cref{l:LASTn:LASTn:transCtxs}, $\evalCtx{E_1}$ does not capture any free names of $\encc{u}{M}$ for any~$u$~($\ast_2$).
                By inversion of typing, $\type{\Gamma} \vdashB \term{\tBfr{\vec{m}}} : \type{S_1}>\type{S}$ where $\type{S_1} = \type{{!}T.S_2}$ ($\ast_3$).
                By \Cref{l:LASTn:LASTn:transBuf}, $\encc{a\>b}{\type{\Gamma} \vdashB \term{\tBfr{\vec{m}}} : \type{S_1}>\type{S}} = \evalCtx[\big]{E_2}[\encc{a\>c}{\type{\emptyset} \vdashB \term{\tBfr{\epsi}} : \type{S_1}>\type{S_1}}]$ ($\ast_4$) and $\encc{a\>b}{\type{\Gamma},\type{\Delta} \vdashB \term{\tBfr{M,\vec{m}}} : \type{S_1}>\type{S}} = \evalCtx[\big]{E_2}[\encc{a\>c}{\type{\Delta} \vdashB \term{\tBfr{M}} : \type{S_1}>\type{S_1}}]$ ($\ast_5$) for some $\evalCtx{E_2},c$.
                Below, we omit types from translations of buffers.
                The thesis holds as follows:
                \begin{align*}
                    & \encc{z}{\pRes{x\tBfr{\vec{m}}y} ( \tCtx{\hat{F}}[\tSend M \, x] \prl C )}
                    \\ \displaybreak[1]
                    &= \pRes{ax} \pRes{by} ( \encc{a\>b}{\tBfr{\vec{m}}} \| \encc{z}{\tCtx{\hat{F}}[\tSend M \, x]} \| \pRes{\_\_} \encc{\_}{C} )
                    \\ \displaybreak[1]
                    &= \pRes{ax} \pRes{by} ( \evalCtx[\big]{E_2}[\encc{a\>c}{\tBfr{\epsi}}] \| \evalCtx[\big]{E_1}[\encc{z'}{\tSend M \, x}] \| \pRes{\_\_} \encc{\_}{C} )
                    \tag{$\ast_1$,$\ast_4$}
                    \\ \displaybreak[1]
                    &= \begin{array}[t]{@{}l@{}}
                        \pRes{ax} \pRes{by} (
                        \\ ~
                        \evalCtx{E_2} \textcolor{clrEvalCtx}{\big[} \pIn a(\_,c_1) ; \pRes{d_1e_1} \big( \pOut c_1[\_,d_1] \| \pIn e_1(f_1,g_1) ; \pRes{h_1k_1} ( \pIn c(\_,l_1) ; \pOut l_1[f_1,h_1] \| \encc{g_1\>k_1}{\tBfr{\epsi}} ) \big) \textcolor{clrEvalCtx}{\big]}
                        \\ ~
                        {} \| \evalCtx{E_1} \textcolor{clrEvalCtx}{\big[}\pRes{a_2b_2} \pRes{c_2d_2} \big( \begin{array}[t]{@{}l@{}}
                            \pIn a_2(\_,e_2) ; \encc{e_2}{M} \| \pOut x[\_,c_2] 
                            \\
                            {} \| \pIn d_2(\_,f_2) ; \pRes{g_2h_2} ( \pOut f_2[b_2,g_2] \| \pOut h_2[\_,z'] ) \big) \textcolor{clrEvalCtx}{\big]}
                        \end{array}
                        \\ ~
                        {} \| \pRes{\_\_} \encc{\_}{C} )
                    \end{array}
                    \tag{$\ast_3$}
                    \\ \displaybreak[1]
                    &\redd \begin{array}[t]{@{}l@{}}
                        \pRes{c_2d_2} \pRes{by} (
                        \\ ~
                        \evalCtx[\big]{E_2}[\pRes{d_1e_1} \big( \pOut c_2[\_,d_1] \| \pIn e_1(f_1,g_1) ; \pRes{h_1k_1} ( \pIn c(\_,l_1) ; \pOut l_1[f_1,h_1] \| \encc{g_1\>k_1}{\tBfr{\epsi}} ) \big)]
                        \\ ~
                        {} \| \evalCtx[\big]{E_1}[\pRes{a_2b_2} \big( \pIn a_2(\_,e_2) ; \encc{e_2}{M} \| \pIn d_2(\_,f_2) ; \pRes{g_2h_2} ( \pOut f_2[b_2,g_2] \| \pOut h_2[\_,z'] ) \big)]
                        \\ ~
                        {} \| \pRes{\_\_} \encc{\_}{C} )
                    \end{array}
                    \\ \displaybreak[1]
                    &\redd \begin{array}[t]{@{}l@{}}
                        \pRes{e_1d_1} \pRes{by} (
                        \\ ~
                        \evalCtx[\big]{E_2}[\pIn e_1(f_1,g_1) ; \pRes{h_1k_1} ( \pIn c(\_,l_1) ; \pOut l_1[f_1,h_1] \| \encc{g_1\>k_1}{\tBfr{\epsi}} )]
                        \\ ~
                        {} \| \evalCtx[\big]{E_1}[\pRes{a_2b_2} \big( \pIn a_2(\_,e_2) ; \encc{e_2}{M} \| \pRes{g_2h_2} ( \pOut d_1[b_2,g_2] \| \pOut h_2[\_,z'] ) \big)]
                        \\ ~
                        {} \| \pRes{\_\_} \encc{\_}{C} )
                    \end{array}
                    \\ \displaybreak[1]
                    &\redd \begin{array}[t]{@{}l@{}}
                        \pRes{b_2a_2} \pRes{g_2h_2} \pRes{by} (
                        \\ ~
                        \evalCtx[\big]{E_2}[\pRes{h_1k_1} ( \pIn c(\_,l_1) ; \pOut l_1[b_2,h_1] \| \encc{g_2\>k_1}{\tBfr{\epsi}} )]
                        \\ ~
                        {} \| \evalCtx[\big]{E_1}[\pIn a_2(\_,e_2) ; \encc{e_2}{M} \| \pOut h_2[\_,z']]
                        \\ ~
                        {} \| \pRes{\_\_} \encc{\_}{C} )
                    \end{array}
                    \\ \displaybreak[1]
                    &\equiv \begin{array}[t]{@{}l@{}}
                        \pRes{ax} \pRes{by} (
                        \\ ~
                        \evalCtx[\big]{E_2}[\pRes{a_2b_2} \pRes{h_1k_1} ( \pIn a_2(\_,e_2) ; \encc{e_2}{M} \| \pIn c(\_,l_1) ; \pOut l_1[b_2,h_1] \| \encc{a\>k_1}{\tBfr{\epsi}} )]
                        \\ ~
                        {} \| \evalCtx[\big]{E_1}[\pOut x[\_,z']]
                        \\ ~
                        {} \| \pRes{\_\_} \encc{\_}{C} )
                    \end{array}
                    \tag{$\ast_2$}
                    \\ \displaybreak[1]
                    &= \pRes{ax} \pRes{by} ( \evalCtx[\big]{E_2}[\encc{a\>c}{\tBfr{M}}] \| \evalCtx[\big]{E_1}[\encc{z'}{x}] \| \pRes{\_\_} \encc{\_}{C} )
                    \\ \displaybreak[1]
                    &= \pRes{ax} \pRes{by} ( \encc{a\>b}{\tBfr{M,\vec{m}}} \| \encc{z}{\tCtx{\hat{F}}[x]} \| \pRes{\_\_} \encc{\_}{C} )
                    \tag{$\ast_1$,$\ast_5$}
                    \\ \displaybreak[1]
                    &= \encc{z}{\pRes{x\tBfr{M,\vec{m}}y} ( \tCtx{\hat{F}}[x] \prl C )}
                \end{align*}
            }
            \proofTransComplSend
            
        \item 
            Rule~\ruleLabel{red-lift-C}: $\term{C} \reddC \term{C'}$ implies $\tCtx{G}[C] \reddC \tCtx{G}[C']$.
            This case follows from \Cref{l:LASTn:LASTn:transCtxs} and the IH.

        \item
            Rule~\ruleLabel{red-lift-M}: $\term{M} \reddM \term{M'}$ implies $\tCtx{F}[M] \reddM \tCtx{F}[M']$.
            By \Cref{l:LASTn:LASTn:transCtxs}, it suffices to show completeness on the level of terms.
            Hence, we apply induction on the derivation of $\term{M} \reddM \term{M'}$ (\ih2).
            \ifappendix In \Cref{as:LASTn:oc:compl}, this property is proven separately as \Cref{t:LASTn:LASTn:transTermRed}.\fi
            Here, we discuss one representative reduction for computation, as well as the structural rules:
            \begin{itemize}

                \item 
                    \global\def\proofTransTermRedLam{%
                        Rule~\ruleLabel{red-lam}:
                        $\term{(\lam x . M)\ N} \reddM \term{M \tSub{ N/x }}$.
                        The thesis holds as follows:
                        \begin{align*}
                            \encc{z}{(\lam x . M)\ N} &= \pRes{a_1b_1} \pRes{c_1d_1} ( \pIn a_1(x,a_2) ; \encc{a_2}{M} \| \pOut b_1[c_1,z] \| \pIn d_1(\_,e_1) ; \encc{e_1}{N} )
                            \\
                            &\redd \pRes{c_1d_1} ( \encc{z}{M} \{ c_1/x \} \| \pIn d_1(\_,e_1) ; \encc{e_1}{N} )
                            \\
                            &\equiv \pRes{xd_1} ( \encc{z}{M} \| \pIn d_1(\_,e_1) ; \encc{e_1}{N} )
                            \\
                            &= \encc{z}{M \tSub{ N/x }}
                        \end{align*}
                    }
                    \proofTransTermRedLam
                    
                \item
                    Rule~\ruleLabel{red-lift}: $\term{M} \reddM \term{N}$ implies $\tCtx{R}[M] \reddM \tCtx{R}[N]$.
                    This case follows from \Cref{l:LASTn:LASTn:transCtxs} and \ih2.
                    
                \item
                    Rule~\ruleLabel{red-lift-sc}: $\term{M} \equivM \term{M'}$, $\term{M'} \reddM \term{N'}$, and $\term{N'} \equivM \term{N}$ imply $\term{M} \reddM \term{N}$.
                    By \ih2, it suffices to show that the translation preserves structural congruence of terms as structural congruence of processes.
                    Hence, we apply induction on the derivation of $\term{M} \equivM \term{M'}$ (and similarly for $\term{N'} \equivM \term{N}$; \ih3).
                    \ifappendix In \Cref{as:LASTn:oc:compl}, this property is proven separately as \Cref{t:LASTn:LASTn:transTermSc}.\fi
                    \global\def\proofLASTnTransTermSc#1#2#3{%
                        The inductive cases follow from \ih#1 and \Cref{l:LASTn:LASTn:transCtxs} directly.
                        We detail the (only) base case of Rule~\ruleLabel{sc-sub-ext}:
                        $\term{x} \notin \fn(\tCtx{R})$ implies $\term{(\tCtx{R}[M]) \tSub{ N/x }} \equivM \tCtx{R}[M \tSub{ N/x }]$.

                        The analysis is by induction on the structure of $\tCtx{R}$ (\ih#2), assuming $\term{x} \notin \fn(\tCtx{R})$.
                        The base case where $\tCtx{R} = \tHole$ is immediate.
                        We detail one representative inductive case: $\tCtx{R} = \term{\tCtx{R'}\ M'}$.
                        The thesis holds as follows:
                        \begin{align*}
                            & \encc{z}{(\tCtx{R'}[M]\ M') \tSub{ N/x }} 
                            \\
                            &= \pRes{xa_1} \big( \begin{array}[t]{@{}l@{}}
                                \pRes{a_2b_2} \pRes{c_2d_2} ( \encc{a_2}{\tCtx{R'}[M]} \| \pOut b_2[c_2,z] \| \pIn d_2(\_,e_2) ; \encc{e_2}{M'} )
                                \\
                                {} \| \pIn a_1(\_,b_1) ; \encc{b_1}{N} \big)
                            \end{array}
                            \\
                            &\equiv \pRes{a_2b_2} \pRes{c_2d_2} \big( \begin{array}[t]{@{}l@{}}
                                \pRes{xa_1} ( \encc{a_2}{\tCtx{R'}[M]} \| \pIn a_1(\_,b_1) ; \encc{b_1}{N} )
                                \\
                                {} \| \pOut b_2[c_2,z] \| \pIn d_2(\_,e_2) ; \encc{e_2}{M'} \big)
                            \end{array}
                            \\
                            &= \pRes{a_2b_2} \pRes{c_2d_2} \big( \encc{a_2}{(\tCtx{R'}[M]) \tSub{ N/x }} \| \pOut b_2[c_2,z] \| \pIn d_2(\_,e_2) ; \encc{e_2}{M'} \big)
                            \\
                            &\equiv \pRes{a_2b_2} \pRes{c_2d_2} \big( \encc{a_2}{\tCtx{R'}[M \tSub{ N/x }]} \| \pOut b_2[c_2,z] \| \pIn d_2(\_,e_2) ; \encc{e_2}{M'} \big)
                            \tag{\ih#2}
                            \\
                            &= \encc{z}{(\tCtx{R'}[M \tSub{ N/x }])\ M'} = \encc{z}{\tCtx{R}[M \tSub{ N/x }]}
                            \tag*{#3}
                        \end{align*}
                    }
                    \proofLASTnTransTermSc34{}

            \end{itemize}
            
        \item
            Rule~\ruleLabel{red-conf-lift-sc}:
            $\term{C} \equivC \term{C'}$, $\term{C'} \reddC \term{D'}$, and $\term{D'} \equivC \term{D}$ imply $\term{C} \reddC \term{D}$.
            By the IH, it suffices to show that the translation preserves structural congruence of configurations as structural congruence of processes.
            Hence, we apply induction on the derivation of $\term{C} \equivC \term{C'}$ (and similarly for $\term{D'} \equivC \term{D}$; \ih2).
            \ifappendix In \Cref{as:LASTn:oc:compl}, this property is proven separately as \Cref{t:LASTn:LASTn:transConfSc}.\fi
            The inductive cases follow from \ih2 and \Cref{l:LASTn:LASTn:transCtxs} straightforwardly.
            We detail the interesting base case of
            \global\def\proofLASTnTransConfScResSwap#1#2{%
                Rule~\ruleLabel{sc-res-swap}:
                $\term{\pRes{x\tBfr{\epsi}y} C} \equivC \term{\pRes{y\tBfr{\epsi}x} C}$.
                Both directions are analogous; we detail the left to right direction.
                We first infer the typing of the left configuration:
                \[
                    \begin{bussproof}
                        \bussAssume{
                            \type{\emptyset} \vdashB \term{\tBfr{\epsi}} : \type{S'}>\type{S}
                        }
                        \bussAssume{
                            \type{\Gamma} \vdashC{\phi} \term{C} : \type{T}
                        }
                        \bussBin[\ruleLabel{typ-res/-buf}]{
                            \type{\Gamma'} \vdashC{\phi} \term{\pRes{x\tBfr{\epsi}y} C} : \type{T}
                        }
                    \end{bussproof}
                \]
                Here, $\type{\Gamma'} = \type{\Gamma} \setminus \term{x}:\type{S'} , \term{y}:\type{\ol{S}}$.
                The analysis depends on whether $\term{x} = \term{\tNil}$ and/or $\term{y} = \term{\tNil}$.
                In each case, we show that
                \begin{align}
                    \encc{a\>b}{\type{\emptyset} \vdashB \term{\tBfr{\epsi}} : \type{S'}>\type{S}} = \encc{b\>a}{\type{\emptyset} \vdashB \term{\tBfr{\epsi}} : \type{\ol{S}}>\type{\ol{S'}}}:
                    \label{eq:LASTn2:trans:confSC:resSwap#2}
                \end{align}
                \begin{itemize}

                    \item
                        Case $\term{x} = \term{y} = \term{\tNil}$, or $\term{x} \neq \term{\tNil}$ and $\term{y} \neq \term{\tNil}$.
                        Either way, $\type{S'} = \type{S}$.
                        If $\type{S'} = \type{\tNil}$, both translations are $\0$; the thesis follows immediately.
                        Otherwise, the thesis follows by induction on the structure of $\type{S'}$; clearly, the resulting translations are exactly the same.

                    \item
                        Case $\term{x} = \term{\tNil}$ and $\term{y} \neq \term{\tNil}$, or $\term{x} \neq \term{\tNil}$ and $\term{y} = \term{\tNil}$.
                        W.l.o.g., assume the former.
                        Then $\type{S'} = \type{\tNil}$ and $\type{S} = \type{\tEnd}$.
                        The thesis then holds as follows:
                        \begin{align*}
                            \encc{a\>b}{\type{\emptyset} \vdashB \term{\tBfr{\epsi}} : \type{\tNil}>\type{\tEnd}} &= \pIn b(\_,c) ; \pOut c[\_,\_]
                            \\
                            &= \encc{b\>c}{\type{\emptyset} \vdashB \term{\tBfr{\epsi}} : \type{\tEnd}>\type{\tNil}}
                        \end{align*}

                \end{itemize}
                The thesis then holds as follows:
                \begin{align*}
                    \encc{z}{\pRes{x\tBfr{\epsi}y} C} &= \pRes{ax} \pRes{by} ( \encc{a\>b}{\type{\emptyset} \vdashB \term{\tBfr{\epsi}} : \type{S'}>\type{S}} \| \encc{z}{C} )
                    \\
                    &= \pRes{ax} \pRes{by} ( \encc{b\>a}{\type{\emptyset} \vdashB \term{\tBfr{\epsi}} : \type{\ol{S}}>\type{\ol{S'}}} \| \encc{z}{C} )
                    \tag*{\eqref{eq:LASTn2:trans:confSC:resSwap#2}}
                    \\
                    &\equiv \pRes{by} \pRes{ax} ( \encc{b\>a}{\type{\emptyset} \vdashB \term{\tBfr{\epsi}} : \type{\ol{S}}>\type{\ol{S'}}} \| \encc{z}{C} )
                    \\
                    &= \encc{z}{\pRes{y\tBfr{\epsi}x} C}
                    \tag*{#1}
                \end{align*}
            }
            \proofLASTnTransConfScResSwap{\qedhere}{}

    \end{itemize}
\end{proof}

\begin{proof}[Proof of soundness (\Cref{t:LASTn:LASTn:transSound})]
    \global\def\proofLASTnTransSoundIntro{%
        By induction on the number $k$ of steps $\encc{z}{C} \redd^k Q$ (\ih1).
        We distinguish cases on all possible initial reductions $\encc{z}{C} \redd Q_0$ and discuss all possible following reductions.
        Here, we rely on \APCP's confluence of independent reductions, allowing us to focus on a specific sequence of reductions, postponing other possibilities that eventually lead to the same result.

        We then use induction on the structure of $\term{C}$ (\ih2).
        The goal is to identify some $\term{D_0}$ such that we can isolate $k_0 \geq 0$ reductions such that $\term{C} \reddC \term{D_0}$ and $\encc{z}{C} \redd Q_0 \redd^{k_0} \encc{z}{D_0}$ (where $k_0$ may be different in each case).
        We then have $\encc{z}{D_0} \redd^{k-k_0} Q$, so it follows from \ih1 that there exists $\term{D}$ such that $\term{D_0} \reddC^\ast \term{D}$ and $\encc{z}{D_0} \redd* \encc{z}{D}$.
    }
    \proofLASTnTransSoundIntro
    
    Here, we detail only the interesting case of an interaction between the translations of a buffer and a contained configuration\ifappendix; \Cref{as:LASTn:oc:sound} details the whole proof\fi.%
    \global\def\proofLASTnTransSoundBuf#1{%
        $\term{C} = \term{\pRes{x\tBfr{\vec{m}}y} C_1}$.
        We have
        \[
            \encc{z}{\pRes{x\tBfr{\vec{m}}y} C_1} = \pRes{a_1x} \pRes{b_1y} ( \encc{a_1\>b_1}{\tBfr{\vec{m}}} \| \encc{z}{C_1} ).
        \]
        The reduction may originate from (i)~$\encc{a_1\>b_1}{\tBfr{\vec{m}}}$, (ii)~$\encc{z}{C_1}$, (iii)~a synchronization between $a_1$ in $\encc{a_1\>b_1}{\tBfr{\vec{m}}}$ and $x$ in $\encc{z}{C_1}$, or (iv)~a synchronization between $b_1$ in~$\encc{a_1\>b_1}{\tBfr{\vec{m}}}$ and $y$ in $\encc{z}{C_1}$.
        We detail each case:
        \begin{itemize}

            \item[(i)]
                The reduction stems from $\encc{a_1\>b_1}{\tBfr{\vec{m}}}$.
                No matter the $\term{\tBfr{\vec{m}}}$, no reduction is possible.

            \item[(ii)]
                The reduction originates from $\encc{z}{C_1}$.
                We thus have $\encc{z}{C_1} \redd Q_1$.
                By \ih2, there are~$\term{D_1},k_1 \geq 0$ such that $\term{C_1} \reddC \term{D_1}$ and $\encc{z}{C_1} \redd Q_1 \redd^{k_1} \encc{z}{D_1}$.
                Let $\term{D_0} \deq \term{\pRes{x\tBfr{\vec{m}}y} D_1}$.
                We have $\term{C} \reddC \term{D_0}$.
                Moreover:
                \begin{align*}
                    \encc{z}{C}
                    &= \pRes{a_1x} \pRes{b_1y} ( \encc{a_1\>b_1}{\tBfr{\vec{m}}} \| \encc{z}{C_1} )
                    \\
                    &\redd^{k_1+1} \pRes{a_1x} \pRes{b_1y} ( \encc{a_1\>b_1}{\tBfr{\vec{m}}} \| \encc{z}{D_1} )
                    \\
                    &= \encc{z}{D_0}
                \end{align*}

            \item[(iii)]
                The reduction originates from a synchronization between $a_1$ in $\encc{a_1\>b_1}{\tBfr{\vec{m}}}$ and $x$ in~$\encc{z}{C_1}$.
                By well-typedness, $\type{\Delta} \vdashB \term{\tBfr{\vec{m}}} : \type{S'}>\type{S}$.
                Note first that, by \Cref{l:LASTn:LASTn:transBuf}, $\type{S'} \neq \type{\tNil}$ implies that there are $\evalCtx{E_2},c_1$ such that \mbox{$\enc{a_1\>b_1}{\term{\tBfr{\vec{m}}} : \type{S'}>\type{S}} = \evalCtx[\big]{E_2}[\enc{a_1\>c_1}{\term{\tBfr{\epsi}} : \type{S'}>\type{S'}}]$}.
                The analysis depends on $\type{S'}$, so we consider all possibilities.
                In each case, if the reduction is indeed possible, we show that the reduction is the first step in the execution of some rule such that $\term{C} \reddC \term{D_0}$.
                The corresponding reduction $\encc{z}{C} \redd Q_0 \redd^{k_0} \encc{z}{D_0}$ follows the corresponding case in the proof of \Cref{t:LASTn:LASTn:transCompl} (\nameref{t:LASTn:LASTn:transCompl}).
                \begin{itemize}

                    \item
                        Case $\type{S'} = \type{\tNil}$.
                        Then $\term{x} = \term{\tNil}$ is not free in $\term{C_1}$, and thus $x$ is not free in $\encc{z}{C_1}$: the reduction is not possible.

                    \item
                        Case $\type{S'} = \type{\tEnd}$.
                        The analysis depends on whether $\type{S} = \type{\tNil}$ or not; w.l.o.g., assume not.
                        We have
                        \[
                            \enc{a_1\>c_1}{\term{\tBfr{\epsi}} : \type{\tEnd}>\type{\tEnd}} = \pIn a_1(\_,c_2) ; \ldots \| \ldots
                        \]
                        Thus, the synchronization is between the receive on $a_1$ and a send on $x$ in~$\encc{z}{C_1}$.

                        A send on a variable $\term{x}$ can only occur in the translation of that variable directly, under some reduction context.
                        Since $\term{x}$ is of type $\type{\tEnd}$ and its translation appears under a reduction context, the only well-typed way for $\term{x}$ to appear in~$\term{C_1}$ is if $\term{C_1} = \tCtx[\big]{G}[\tCtx{F}[\tClose x ; M_1]]$.
                        We then have
                        \[
                            \term{C} \equivC \tCtx[\big]{G'}[\pRes{x\tBfr{\vec{m}}y} ( \tCtx{F}[\tClose x ; M_1] \| \term{C_2} )].
                        \]
                        Hence, the observed reduction is the first step of executing Rule~\ruleLabel{red-close}.

                    \item
                        Case $\type{S'} = \type{{!}T_2.S'_2}$.
                        We have
                        \[
                            \enc{a_1\>c_1}{\term{\tBfr{\epsi}} : \type{{!}T_2.S'_2}>\type{{!}T_2.S'_2}} = \pIn a_1(\_,c_2) ; \ldots
                        \]
                        Thus, the synchronization is between the receive on $a_1$ and a send on $x$ in~$\encc{z}{C_1}$.

                        A send on a variable $\term{x}$ can only occur in the translation of that variable directly, under some reduction context.
                        Since $\term{x}$ is of type $\type{{!}T_2.S'_2}$ and its translation appears under a reduction context, the only well-typed way for $\term{x}$ to appear in~$\term{C_1}$ is if $\term{C_1} = \tCtx[\big]{G}[\tCtx{F}[\tSend M_1 \, x]]$.
                        We then have
                        \[
                            \term{C} \equivC \tCtx[\big]{G'}[\pRes{x\tBfr{\vec{m}}y} ( \tCtx{\hat{F'}}[\tSend M_1 \, x] \| \term{C_2} )].
                        \]
                        Hence, the observed reduction is the first step of executing Rule~\ruleLabel{red-send}.

                    \item
                        Case $\type{S'} = \type{{?}T_2.S'_2}$.
                        We have
                        \[
                            \enc{a_1\>c_1}{\term{\tBfr{\epsi}} : \type{{?}T_2.S'_2}>\type{{?}T_2.S'_2}} = \enc{c_1\>a_1}{\term{\tBfr{\epsi}} : \type{{!}T_2.\ol{S'_2}}>\type{{!}T_2.\ol{S'_2}}} = \pIn c_1(\_,c_2) ; \ldots
                        \]
                        Thus, the reduction is not possible.

                    \item
                        Case $\type{S'} = \type{\oplus \{ i . S^i_2 \}_{i \in I}}$.
                        We have
                        \[
                            \enc{a_1\>c_1}{\term{\tBfr{\epsi}} : \type{\oplus \{ i . S^i_2 \}_{i \in I}}>\type{\oplus \{ i . S^i_2 \}_{i \in I}}} = \pIn a_1(\_,c_2) ; \ldots
                        \]
                        Thus, the synchronization is between the receive on $a_1$ and a send on $x$ in~$\encc{z}{C_1}$.

                        A send on a variable $\term{x}$ can only occur in the translation of that variable directly, under some reduction context.
                        Since $\term{x}$ is of type $\type{\oplus \{ i . S^i_2 \}_{i \in I}}$ and its translation appears under a reduction context, the only well-typed way for $\term{x}$ to appear in~$\term{C_1}$ is if $\term{C_1} = \tCtx[\big]{G}[\tCtx{F}[\tSel j \, x]]$ where $j \in I$.
                        We then have
                        \[
                            \term{C} \equivC \tCtx[\big]{G'}[\pRes{x\tBfr{\vec{m}}y} ( \tCtx{F}[\tSel j \, x] \| \term{C_2} )].
                        \]
                        Hence, the observed reduction is the first step of executing Rule~\ruleLabel{red-select}.

                    \item
                        Case $\type{S'} = \type{\& \{ i . S^i_2 \}_{i \in I}}$.
                        We have
                        \begin{align*}
                            \enc{a_1\>c_1}{\term{\tBfr{\epsi}} : \type{\& \{ i . S^i_2 \}_{i \in I}}>\type{\& \{ i . S^i_2 \}_{i \in I}}}
                            &= \enc{c_1\>a_1}{\term{\tBfr{\epsi}} : \type{\oplus \{ i . \ol{S^i_2} \}_{i \in I}}>\type{\oplus \{ i . \ol{S^i_2} \}_{i \in I}}}
                            \\
                            &= \pIn c_1(\_,c_2) ; \ldots
                        \end{align*}
                        Thus, the reduction is not possible.

                \end{itemize}

            \item[(iv)]
                The reduction originates from a synchronization between $b_1$ in $\encc{a_1\>b_1}{\tBfr{\vec{m}}}$ and $y$ in $\encc{z}{C_1}$.
                By well-typedness, $\type{\Delta} \vdashB \term{\tBfr{\vec{m}}} : \type{S'}>\type{S}$.
                The analysis depends on $\type{S}$, so we consider all possibilities.
                In each case, if the reduction is indeed possible, we show that the reduction is the first step in the execution of some rule such that $\term{C} \reddC \term{D_0}$.
                The corresponding reduction $\encc{z}{C} \redd Q_0 \redd^{k_0} \encc{z}{D_0}$ follows the corresponding case in the proof of \Cref{t:LASTn:LASTn:transCompl} (\nameref{t:LASTn:LASTn:transCompl}).
                \begin{itemize}

                    \item
                        Case $\type{S} = \type{\tNil}$.
                        Then $\term{y} = \term{\tNil}$ is not free in $\term{C_1}$, and thus $y$ is not free in $\encc{z}{C_1}$: the reduction is not possible.

                    \item
                        Case $\type{S} = \type{\tEnd}$.
                        By well-typedness, then $\term{\vec{m}} = \term{\epsi}$.
                        Let $\term{C'} \deq \term{\pRes{y\tBfr{\epsi}x} C_1}$; we have $\term{C} \equivC \term{C'}$ and $\encc{z}{C} \equiv \encc{z}{C'}$\ifappendix (by \Cref{t:LASTn:LASTn:transConfSc})\fi.
                        The thesis then follows as in the analogous case under Sub-case~(iii) above.

                    \item
                        Case $\type{S} = \type{{!}T_2.S_2}$.
                        By well-typedness, then $\term{\vec{m}} = \term{\vec{m}',M_1}$.
                        We have
                        \[
                            \enc{a_1\>b_1}{\term{\tBfr{\vec{m}',M_1}} : \type{S'}>\type{{!}T_2.S_2}} = \pRes{\ldots} \pRes{\ldots} ( \ldots \| \pIn b_1(\_,h_2) ; \ldots \| \ldots ).
                        \]
                        Thus, the synchronization is between the receive on $b_1$ and a send on $y$ in~$\encc{z}{C_1}$.

                        A send on a variable $\term{y}$ can only occur in the translation of that variable directly, under some reduction context.
                        Since $\term{y}$ is of type $\type{\ol{S}} = \type{{?}T_2.\ol{S_2}}$ and its translation appears under a reduction context, the only well-typed way for $\term{y}$ to appear in~$\term{C_1}$ is if $\term{C_1} = \tCtx[\big]{G}[\tCtx{F}[\tRecv y]]$.
                        We then have
                        \[
                            \term{C} \equivC \tCtx[\big]{G'}[\pRes{x\tBfr{\vec{m}',M_1}y} ( \tCtx{\hat{F'}}[\tRecv y] \| \term{C_2} )].
                        \]
                        Hence, the observed reduction is the first step of executing Rule~\ruleLabel{red-recv}.

                    \item
                        Case $\type{S} = \type{{?}T_2.S_2}$.
                        By well-typedness, $\term{\vec{m}} = \term{\epsi}$; this case is analogous to Case~\mbox{$\type{S} = \type{\tEnd}$}.

                    \item
                        Case $\type{S} = \type{\oplus \{ i : S_2^i \}_{i \in I}}$.
                        By well-typedness, then $\term{\vec{m}} = \term{\vec{m}',j}$ where $j \in I$.
                        We have
                        \[
                            \enc{a_1\>b_1}{\term{\tBfr{\vec{m}',j}} : \type{S'}>\type{\oplus \{ i : S_2^i \}_{i \in I}}} = \pRes{\ldots} ( \pIn b_1(\_,e_2) ; \ldots \| \ldots ).
                        \]
                        Thus, the synchronization is between the receive on $b_1$ and a send on $y$ in~$\encc{z}{C_1}$.

                        A send on a variable $\term{y}$ can only occur in the translation of that variable directly, under some reduction context.
                        Since $\term{y}$ is of type $\type{\ol{S}} = \type{\& \{ i : \ol{S_2^i} \}_{i \in I}}$ and its translation appears under a reduction context, the only well-typed way for~$\term{y}$ to appear in $\term{C_1}$ is if $\term{C_1} = \tCtx[\big]{G}[\tCtx{F}[\tCase y \tOf \{ i : M_{1.i} \}_{i \in I}]]$.
                        We then have
                        \[
                            \term{C} \equivC \tCtx[\big]{G'}[\pRes{x\tBfr{\vec{m}',j}y} ( \tCtx{F}[\tCase y \tOf \{ i : M_{1.i} \}_{i \in I}] \| \term{C_2} )].
                        \]
                        Hence, the observed reduction is the first step of executing Rule~\ruleLabel{red-case}.

                    \item
                        Case $\type{S} = \type{\& \{ i : S_2^i \}_{i \in I}}$.
                        By well-typedness, then $\term{\vec{m}} = \term{\epsi}$; this case is analogous to Case~$\type{S} = \type{\tEnd}$ above.
                        #1

                \end{itemize}

        \end{itemize}
    }
    That is, \proofLASTnTransSoundBuf{\qedhere}
\end{proof}


\subsection{Deadlock-free \LASTn}
\label{s:LASTn:LASTn:df}

By virtue of \ifappendix \Cref{t:LASTn:APCP:df}\else \APCP's deadlock-freedom result (\Cref{t:APCP:df} on \Cpageref{t:APCP:df})\fi,
well-typed \APCP processes that are typable under empty contexts ($\vdash P \typInf \emptyset$) are deadlock-free.
We may transfer this result to \LASTn configurations by appealing to the operational correctness of our translation (\Cref{t:LASTn:LASTn:transCompl,t:LASTn:LASTn:transSound} above).
Each deadlock-free configuration in \LASTn obtained via this transference satisfies two requirements:
\begin{itemize}

    \item
        The configuration is typable $\type{\emptyset} \vdashC{\tMain} \term{C} : \type{\1}$, i.e.,  it needs no external resources and has no external behavior.

    \item
        The typed translation of the configuration satisfies  priority requirements in \APCP: it is well-typed under `$\vdash$', not only under `$\vdash*$' (cf.\ \Cref{s:LASTn:LASTn:trans}).

\end{itemize}
We rely on soundness (\Cref{t:LASTn:LASTn:transSound}) to transfer deadlock-freedom to configurations:

\begin{theorem}[Deadlock-freedom for \protect\LASTn]
\label{t:LASTn:LASTn:df}
    Given $\type{\emptyset} \vdashC{\tMain} \term{C}: \type{\1}$, if $\vdash \encc{z}{C} \typInf \Gamma$ for some $\Gamma$, then $\term{C} \equiv \term{\tMain\,()}$ or $\term{C} \reddC \term{D}$ for some~$\term{D}$.
\end{theorem}

\begin{proof}
    By \Cref{t:LASTn:LASTn:transTp} (type preservation for the translation), $\Gamma = z:\enct{\1} = z:\bullet$.
    Then \mbox{$\vdash \pRes{z\_} \encc{z}{C} \typInf \emptyset$}.
    Hence, by \ifthesis\ifappendix \Cref{t:LASTn:APCP:df} \else \Cref{t:APCP:df} on \Cpageref{t:APCP:df} \fi\else \Cref{t:APCP:df} \fi (deadlock-freedom), either \mbox{$\pRes{z\_} \encc{z}{C} \equiv \0$} or $\pRes{z\_} \encc{z}{C} \redd Q$ for some $Q$.
    The rest of the analysis depends on which possibility holds:
    \begin{itemize}

        \item
            We have $\pRes{z\_} \encc{z}{C} \equiv \0$.
            We straightforwardly deduce from the well-typedness and translation of $\term{C}$ that $\term{C} \equivC \term{\tMain\,()}$, proving the thesis.

        \item
            We have $\pRes{z\_} \encc{z}{C} \redd Q$ for some $Q$.
            We argue that this implies that $\encc{z}{C} \redd Q_0$, for some~$Q_0$.

            \begin{itemize}

                \item
                    The reduction involves $z$.
                    Since $z$ is of type $\bullet$ in $\encc{z}{C}$, then $z$ must occur in a forwarder, and the reduction involves a forwarder $\pFwd [z<>x]$ for some $x$.
                    Since $x$ is not free in $\encc{z}{C}$, there must be a restriction on $x$.
                    Hence, the forwarder $\pFwd [z<>x]$ can also reduce with this restriction, instead of with the restriction $\pRes{z\_}$.
                    This means that $\encc{z}{C} \redd Q_0$ for some $Q_0$.

                \item
                    The reduction does not involve  $z$, in which case the reduction must be internal to $\encc{z}{C}$.
                    That is, $\encc{z}{C} \redd Q_0$ for some $Q_0$ and $Q \equiv \pRes{z\_}Q_0$.
            \end{itemize}

            First, by contradiction, we show that $z$ is not involved in this reduction:
            Since $z$ is of type $\bullet$ in $\encc{z}{C}$, then $z$ must occur in a forwarder, and the reduction involves a forwarder $\pFwd [z<>x]$ for some $x$.
            However, our translation (\Cref{f:LASTn:LASTn:trans1,f:LASTn:LASTn:trans2,f:LASTn:LASTn:trans3}) does not generate forwarders, so this is clearly a contradiction.
            Hence, the reduction does not involve $z$.

            Thus, the reduction must be internal to $\encc{z}{C}$.
            That is, $\encc{z}{C} \redd Q_0$ for some $Q_0$ and $Q \equiv \pRes{z\_}Q_0$.
            Then, by \Cref{t:LASTn:LASTn:transSound} (soundness), there exists $\term{D}$ such that $\term{C} \reddC^\ast \term{D}$.
            Looking at the proof of \Cref{t:LASTn:LASTn:transSound}\ifappendix (in \Cref{as:LASTn:oc:sound})\fi, it is easy to see that in fact $\term{C} \reddC^+ \term{D}$.
            That is, there exists $\term{D_0}$ such that $\term{C} \reddC \term{D_0}$, proving the thesis.
            \qedhere

    \end{itemize}
\end{proof}

\begin{example}
    \label{x:LASTn:LASTn:deadlocked}
    To illustrate the insights gained from deadlock analysis in \APCP through translation from \LASTn, we consider a seemingly deadlock-free configuration that is not.
    To this end, we first discuss how we can assign priorities to the translated types of the configuration, and extract from the translation requirements on those priorities.
    Next, we show that our example leads to unsatisfiable priority requirements, indicating the elusive deadlock in the original configuration.
    We close the example by considering an alternative configuration that is indeed deadlock-free.

    Our example is a classic showcase of deadlock due to cyclic conncections under synchronous communication.
    Though one would expect this example to be deadlock-free under asynchronous communication, it is not: there is a deadlock induced by the call-by-name semantics of \LASTn.
    Let
    \begin{align*}
        \term{M_{x,y}} &:= \term{
            \begin{array}[t]{@{}l@{}}
                \tLet x_1 = \tSend () \, x \tIn
                \\
                \tLet (v,y_1) = \tRecv y \tIn
                \\
                \tClose x_1 ; \tClose y_1 ; v \normalcolor,
            \end{array}
        }
        &
        \term{C} &:= \term{
            \tMain \, \begin{array}[t]{@{}l@{}}
                \tLet (x,x') = \tNew \tIn
                \\
                \tLet (y,y') = \tNew \tIn
                \\
                \tSpawn M_{x,y} ; M_{y',x'} \normalcolor.
            \end{array}
        }
    \end{align*}
    Before showing the deadlock and its source, we analyze the translation into \APCP.

    By \Cref{t:LASTn:LASTn:transTp} and assigning priority variables to each connective, we have
    \begin{align*}
        &\vdash* \encc{z_1}{M_{x,y}} \typInf \begin{array}[t]{@{}l@{}}
            y : \bullet \tensor^{\pri_1} (\bullet \tensor^{\pri_2} \bullet) \parr^{\pri_3} \bullet \tensor^{\pri_4} \bullet \parr^{\pri_5} \bullet ,
            \\
            x : \bullet \tensor^{\pri_6} \bullet \parr^{\pri_7} (\bullet \parr^{\pri_8} \bullet) \tensor^{\pri_9} \bullet \tensor^{\pri_{10}} \bullet \parr^{\pri_{11}} \bullet , z_1 : \bullet,
        \end{array}
        \\
        &\vdash* \encc{z_2}{M_{y',x'}} \typInf \begin{array}[t]{@{}l@{}}
            x' : \bullet \tensor^{\pi_1} (\bullet \tensor^{\pi_2} \bullet) \parr^{\pi_3} \bullet \tensor^{\pi_4} \bullet \parr^{\pi_5} \bullet ,
            \\
            y' : \bullet \tensor^{\pi_6} \bullet \parr^{\pi_7} (\bullet \parr^{\pi_8} \bullet) \tensor^{\pi_9} \bullet \tensor^{\pi_{10}} \bullet \parr^{\pi_{11}} \bullet , z_2 : \bullet.
        \end{array}
        \\
        &\vdash* \encc{z}{C} \typInf z : \bullet
    \end{align*}
    To determine requirements on these priorities, we analyze the translations of typing rules into \APCP derivations in \Cref{f:LASTn:LASTn:trans1,f:LASTn:LASTn:trans2,f:LASTn:LASTn:trans3} and the priority requirements induced by them.
    Indeed, a deadlock is detected, because \APCP requires $\pri_6 < \pi_1 < \pi_6 < \pri_1 < \pri_6$, which is unsatisfiable.
    
    Before detailing the  origin of these requirements, we  first give an intuition behind this unsatisfiable chain of priority requirements, and how they reflect the deadlock that is indeed present in $\term{C}$.
    From left to right, the requirements denote that (1)~the send on $x$ occurs before the receive on $x'$, (2)~the receive on $x'$ occurs before the send on $y'$, (3)~the send on $y'$ occurs before the receive on $y$, (4)~the receive on $y$ occurs before the send on $x$.
    Steps~(2) and~(3) stand out: $\term{M_{x,y}}$ and $\term{M_{y',x'}}$ seem to define an opposite order.
    This is because in, e.g., $\term{M_{x,y}}$, the call-by-name semantics of \LASTn transforms the $\term{\tLet x_1}$ into an explicit substitution that can only be resolved---thus enabling the $\term{\tSend () \, x}$---once the $\term{\tClose x_1}$ is unblocked by the $\term{\tLet (v,y_1)}$, which in turn is waiting for the $\term{\tRecv y}$ to be resolved. 
    To be precise:
    \begin{align*}
        \term{M_{x,y}} &\reddM \term{\big( \tLet (v,y_1) = \tRecv y \tIn \ldots \big) \tSub{\tSend () \, x / x_1}} =: \term{M'_{x,y}} \not\mkern-6mu\reddM
        \\
        \term{M_{y',x'}} &\reddM \term{\big( \tLet (v',x'_1) = \tRecv x' \tIn \ldots \big) \tSub{\tSend () \, y' / y'_1}} =: \term{M'_{y',x'}} \not\mkern-6mu\reddM
        \\
        \term{C} &\reddC^9 \term{\pRes{s\tBfr{\epsi}s'} \pRes{t\tBfr{\epsi}t'} \big( \tChild \, M'_{x,t} \tSub{s/x} \prl \tMain \, M'_{y',s'} \tSub{t'/y'} \big) } \not\mkern-6mu\reddC
    \end{align*}
    
    We now detail the origin of the conflicting priority requirements; that is, we spell out the rules of the translation and the induced requirements.
    Note that many of these rules include names and priority variables that are not visible outside the derivation, being created from within certain parts of the translation.
    \begin{enumerate}
            
        \item
            Requirement $\pri_6 < \pi_1$ originates from the translation of $\term{\tNew}$ in the $\term{\tLet (x,x')}$.
            More precisely, the translation of $\term{\tNew}$ prepares the translation of the upcoming buffer, which follows the involved types.
            The requirement is induced by a step in the translation of this buffer (Rule~\ruleLabel{typ-buf} in \Cref{f:LASTn:LASTn:trans2,f:LASTn:LASTn:trans3}):
            \[
                \begin{bussproof}
                    \bussAssume{
                        \vdash* \enc{g_3\>k_3}{\term{\tBfr{\epsi}} : \type{\tEnd}>\type{\tEnd}} \typInf \begin{array}[t]{@{}l@{}}
                            g_3 : \bullet \parr^{\rho_5} \bullet \tensor^{\pri_{11}} \bullet , 
                            \\
                            k_3 : \bullet \parr^{\rho_6} \bullet \tensor^{\pi_5} \bullet
                        \end{array}
                    }
                    \bussAssume{
                        \begin{array}[t]{@{}l@{}}
                            \pri_8=\pi_2 , \underline{\pri_6<\pi_1} , \pri_9<\pi_1 ,
                            \\
                            \pi_1 < \pri_8 , \pi_1 < \pi_5 , \pri_7 < \pri_9 ,
                            \\
                            \pi_3 < \pi_2 , \pi_3 < \pi_4
                        \end{array}
                    }
                    \bussBin{
                        \vdash* \enc{d_2\>e_2}{\term{\tBfr{\epsi}} : \type{{!}\1.\tEnd}>\type{{!}\1.\tEnd}} \typInf \begin{array}[t]{@{}l@{}}
                            d_2 : \bullet \parr^{\pri_6} \bullet \tensor^{\pri_7} (\bullet \tensor^{\pri_8} \bullet) \parr^{\pri_9} \bullet \parr^{\pri_{10}} \bullet \tensor^{\pri_{11}} \bullet ,
                            \\
                            e_2 : \bullet \parr^{\pi_1} (\bullet \parr^{\pi_2} \bullet) \tensor^{\pi_3} \bullet \parr^{\pi_4} \bullet \tensor^{\pi_5} \bullet
                        \end{array}
                    }
                \end{bussproof}
            \]
            
        \item 
            Requirement $\pi_1 < \pi_6$ originates from four substeps, including intermediate priorities: $\pi_1 < \pi_3 < \pi_{14} < \pi_{12} < \pi_6$.
            \begin{enumerate}
                    
                \item
                    $\pi_1 < \pi_3$ stems from the translation of the variable $\term{x'}$ (Rule~\ruleLabel{typ-var} in \Cref{f:LASTn:LASTn:trans1}):
                    \[
                        \begin{bussproof}
                            \bussAssume{
                                \underline{\pi_1<\pi_3}
                            }
                            \bussUn{
                                \vdash* \encc{a'_4}{x'} \typInf x' : \bullet \tensor^{\pi_1} (\bullet \tensor^{\pi_2} \bullet) \parr^{\pi_3} \bullet \tensor^{\pi_4} \bullet \parr^{\pi_5} \bullet , a'_4 : (\bullet \parr^{\pi_2} \bullet) \tensor^{\pi_3} \bullet \parr^{\pi_4} \bullet \tensor^{\pi_5} \bullet
                            }
                        \end{bussproof}    
                    \]
                    
                \item
                    $\pi_3 < \pi_{14}$ stems from the translation of the $\term{\tRecv x'}$ (Rule~\ruleLabel{typ-recv} in \Cref{f:LASTn:LASTn:trans1}):
                    \[
                        \begin{bussproof}
                            \bussAssume{
                                \vdash* \encc{a'_4}{x'} \typInf \begin{array}[t]{@{}l@{}}
                                    x' : \bullet \tensor^{\pi_1} (\bullet \tensor^{\pi_2} \bullet) \parr^{\pi_3} \bullet \tensor^{\pi_4} \bullet \parr^{\pi_5} \bullet , 
                                    \\
                                    a'_4 : (\bullet \parr^{\pi_2} \bullet) \tensor^{\pi_3} \bullet \parr^{\pi_4} \bullet \tensor^{\pi_5} \bullet
                                \end{array}
                            }
                            \bussAssume{
                                \begin{array}[t]{@{}l@{}}
                                    \pi_{14}<\pi_2 , \pi_{14}<\pi_{15} ,
                                    \\
                                    \pi_4<\pi_5 ,  \underline{\pi_3<\pi_{14}}
                                \end{array}
                            }
                            \bussBin{
                                \vdash* \encc{b'_3}{\tRecv x'} \typInf \begin{array}[t]{@{}l@{}}
                                    x' : \bullet \tensor^{\pi_1} (\bullet \tensor^{\pi_2} \bullet) \parr^{\pi_3} \bullet \tensor^{\pi_4} \bullet \parr^{\pi_5} \bullet , 
                                    \\
                                    b'_3 : (\bullet \parr^{\pi_2} \bullet) \tensor^{\pi_{14}} \bullet \parr^{\pi_{15}} \bullet \tensor^{\pi_5} \bullet
                                \end{array}
                            }
                        \end{bussproof}
                    \]
                    
                \item
                    $\pi_{14} < \pi_{12}$ stems from the translation of $\term{\tLet (v',x'_1)}$ (Rule~\ruleLabel{typ-split} in \Cref{f:LASTn:LASTn:trans1}):
                    \[
                        \begin{bussproof}
                            \def\defaultHypSeparation{\hskip4mm}
                            \bussAssume{
                                \begin{array}[b]{@{}l@{}}
                                    \vdash* \encc{b'_3}{\tRecv x'} 
                                    \\
                                    \typInf \begin{array}[t]{@{}l@{}}
                                        x' : \bullet \tensor^{\pi_1} (\bullet \tensor^{\pi_2} \bullet) \parr^{\pi_3} \bullet \tensor^{\pi_4} \bullet \parr^{\pi_5} \bullet ,
                                        \\
                                        b'_3 : (\bullet \parr^{\pi_2} \bullet) \tensor^{\pi_{14}} \bullet \parr^{\pi_{15}} \bullet \tensor^{\pi_5} \bullet
                                    \end{array}
                                \end{array}
                            }
                            \bussAssume{
                                \begin{array}[b]{@{}l@{}}
                                    \vdash* \encc{a'_2}{\ldots}
                                    \\
                                    \typInf \begin{array}[t]{@{}l@{}}
                                        y'_1 : \bullet \tensor^{\pi_{12}} \bullet \parr^{\pi_{11}} \bullet ,
                                        \\
                                        v' : \bullet \tensor^{\pi_2} \bullet ,
                                        \\
                                        x'_1 : \bullet \tensor^{\pi_{15}} \bullet \parr^{\pi_5} \bullet
                                    \end{array}
                                \end{array}
                            }
                            \bussAssume{
                                 \underline{\pi_{14}<\pi_{12}}
                            }
                            \bussTern{
                                \vdash* \encc{a'_2}{\tLet (v',x'_1) = \tRecv x' \tIn \ldots} \typInf \begin{array}[t]{@{}l@{}}
                                    x' : \bullet \tensor^{\pi_1} (\bullet \tensor^{\pi_2} \bullet) \parr^{\pi_3} \bullet \tensor^{\pi_4} \bullet \parr^{\pi_5} \bullet ,
                                    \\
                                    y'_1 : \bullet \tensor^{\pi_{12}} \bullet \parr^{\pi_{11}} \bullet , a'_2 : \bullet
                                \end{array}
                            }
                        \end{bussproof}
                    \]
                    
                \item
                    Finally, $\pi_{12} < \pi_6$ originates from the translation of the $\term{\tLet y'_1}$.
                    To be precise, this syntactic sugar breaks down into an abstraction and an application; the requirement originates from the application (Rule~\ruleLabel{typ-app} in \Cref{f:LASTn:LASTn:trans1}):
                    \[
                        \begin{bussproof}
                            \def\defaultHypSeparation{\hskip1.5mm}
                            \def\ScoreOverhang{1pt}
                            \bussAssume{
                                \begin{array}[b]{@{}l@{}}
                                    \vdash* \encc{a'_1}{\lam y'_1 \ldots} 
                                    \\
                                    \typInf \begin{array}[t]{@{}l@{}}
                                        x' : \bullet \tensor^{\pi_1} (\bullet \tensor^{\pi_2} \bullet) \parr^{\pi_3} \bullet \tensor^{\pi_4} \bullet \parr^{\pi_5} \bullet ,
                                        \\
                                        a'_1 : (\bullet \tensor^{\pi_{12}} \bullet \parr^{\pi_{11}} \bullet) \parr^{\pi_{13}} \bullet
                                    \end{array}
                                \end{array}
                            }
                            \bussAssume{
                                \begin{array}[b]{@{}l@{}}
                                    \vdash* \encc{e'_1}{\tSend () \, y'} 
                                    \\
                                    \typInf \begin{array}[t]{@{}l@{}}
                                        y' : \begin{array}[t]{@{}l@{}}
                                            \bullet \tensor^{\pi_6} \bullet \parr^{\pi_7} (\bullet \parr^{\pi_8} \bullet)
                                            \\
                                            {} \tensor^{\pi_9} \bullet \tensor^{\pi_{10}} \bullet \parr^{\pi_{11}} \bullet ,
                                        \end{array}
                                        \\
                                        e'_1 : \bullet \tensor^{\pi_{11}} \bullet
                                    \end{array}
                                \end{array}
                            }
                            \bussAssume{
                                \begin{array}[b]{@{}l@{}}
                                    \underline{\pi_{12}<\pi_6} ,
                                    \\
                                    \pi_{13}<\pi_{12}
                                \end{array}
                            }
                            \bussTern{
                                \vdash* \encc{z_2}{\tLet y'_1 = \tSend y' \tIn \ldots} \typInf \begin{array}[t]{@{}l@{}}
                                    x' : \bullet \tensor^{\pi_1} (\bullet \tensor^{\pi_2} \bullet) \parr^{\pi_3} \bullet \tensor^{\pi_4} \bullet \parr^{\pi_5} \bullet ,
                                    \\
                                    y' : \bullet \tensor^{\pi_6} \bullet \parr^{\pi_7} (\bullet \parr^{\pi_8} \bullet) \tensor^{\pi_9} \bullet \tensor^{\pi_{10}} \bullet \parr^{\pi_{11}} \bullet ,
                                    \\
                                    z_2 : \bullet  
                                \end{array}
                            }
                        \end{bussproof}
                    \]
                
            \end{enumerate}
            
        \item
            Requirement $\pi_6 < \pri_1$ originates from the buffer prepared by the translation of $\term{\tNew}$ in the $\term{\tLet (y,y')}$.
            It is derived similar to Requirement~(1).

        \item 
            Requirement $\pri_1 < \pri_6$ originates from four similar substeps as in Requirement~(2).

    \end{enumerate}
   We conclude that the translation of $\term{C}$ cannot be typed under $\vdash$, because the priority requirements cannot be satisfied.
    Hence, \Cref{t:LASTn:LASTn:df} does not apply, and so deadlock-freedom cannot be guaranteed for $\term{C}$.

\smallskip

Note that it is straightforward to define variants of $\term{C}$ whose translations are well-typed under $\vdash$, thus guaranteeing deadlock-freedom through \Cref{t:LASTn:LASTn:df}.
An example is the variant of $ \term{M_{x,y}}$ above in which the $\term{\tClose x_1}$ occurs before the $\term{\tLet (v,y_1)}$.
\end{example}

\begin{remark}
    \label{r:LASTn:LASTn:transDfComplex}
    Determining the deadlock-freedom of well-typed \LASTn programs by means of translation into \APCP is of similar complexity as a direct, priority-based approach.
    It is easy to see that the translation is $\mcl{O}(n)$ in the size of the \LASTn typing derivation.
    We could add priorities to the types of \LASTn and derive priority requirements from the translation; this would boil down to roughly the same amount of checking.
    Hence, the indirect approach via \APCP is similarly complex.
    This justifies keeping the type system of \LASTn simple, compared to type systems that use priorities for deadlock-free functional programs such as those by, e.g., Padovani and Novara~\cite{conf/forte/PadovaniN15} and Kokke and Dardha~\cite{conf/forte/KokkeD21}.
\end{remark}


\section{Related Work}
\label{s:APCP:rw}

Closely related work on the \ifAPCP theme \else two themes \fi of the \ifthesis chapter\else paper\fi (deadlock-freedom by typing\ifAPCP\else and functional calculi with concurrency\fi) has been already discussed throughout the \ifthesis chapter\else paper\fi.
Here we comment on other related literature.

\paragraph{Asynchronous communication.}
Asynchronous communication has a longstanding history in process algebras (see, e.g., \cite{conf/concurrency/BergstraKT85,conf/ifip/HeJH90,conf/lics/deBoerKP92}).
The first accounts of asynchronous communication for the $\pi$-calculus were developed independently by Honda and Tokoro~\cite{conf/ecoop/HondaT91,conf/occ/HondaT91} and by Boudol~\cite{report/Boudol92}.
Palamidessi~\cite{journal/mscs/Palamidessi03} shows that the synchronous $\pi$-calculus is strictly more expressive than its asynchronous variant, due to \emph{mixed choices}: non-deterministic choices involving both inputs and outputs.
Beauxis \etal~\cite{chapter/BeauxisPV08} study the exact form of asynchronous communication modeled by the asynchronous $\pi$-calculus; they examine communication mediated by different mechanisms (bags, stacks, queues) in the synchronous $\pi$-calculus, and prove that bags lead to the strongest operational correspondence with the asynchronous $\pi$-calculus.

As discussed already, asynchrony is a relevant phenomenon in session $\pi$-calculi; the communication structures delineated by sessions strongly influence the expected ordering of messages.
The expressiveness gap between asynchronous and synchronous communication shown by Palamidessi in the untyped setting does not hold in this context, since session-typed $\pi$-calculi consider only deterministic choices and do not account for mixed choices.
A notable exception is the work on \emph{mixed sessions} by Casal \etal~\cite{conf/esop/VasconcelosCAM20,journal/tcs/CasalMV22}, which does not address deadlock-freedom.
Nevertheless, fundamental differences between mixed choices in untyped and session-typed settings remain, as established by Peters and Yoshida~\cite{conf/express/PetersY22}.

In the context of session types, the first formal theory of asynchronous communication for the $\pi$-calculus is by Kouzapas \etal~\cite{conf/forte/KouzapasYH11}, using a buffered semantics.
Their focus is on the behavioral theory induced by asynchrony and the program transformations it enables.
Follow-up work~\cite{journal/mscs/KouzapasYHH16} goes beyond to consider asynchronous communication in combination with constructs for event-driven programming, and develops a corresponding type system and behavioral theory.
Unlike our work, these two works do not consider deadlock-freedom for asynchronous session processes.

DeYoung \etal~\cite{conf/csl/DeYoungCPT12} give the first connection between linear logic and a session-typed $\pi$-calculus with asynchronous communication using a continuation-passing semantics, and show that this semantics is equivalent to a buffered semantics.
As mentioned before, \APCP's semantics is based on this work, except that our typing rules for sending and selection are axiomatic and that we consider recursion.
On a similar line, Jia \etal~\cite{conf/popl/JiaGP16} consider a session-typed language with asynchronous communication; their focus is on the dynamic monitoring of session behaviors, not on deadlock-freedom.
Both these works are based on the correspondence with intuitionistic linear logic, which restricts the kind of process networks allowed by typing.
Pruiksma and Pfenning~\cite{conf/places/PruiksmaP19,journal/jlamp/PruiksmaP20} derive a ``propositions-as-sessions'' type system from adjoint logic, which combines several logics with different structural rules through modalities~\cite{report/Reed09,report/PruiksmaCPR18}.
Their process language features asynchronous communication with continuation-passing and their type system treats asynchronous, non-blocking outputs via axiomatic typing rules, similar to Rules~\ruleLabel{typ-send} and~\ruleLabel{typ-sel} in \Cref{f:APCP:typingRules}.

Padovani's linear type system for the asynchronous $\pi$-calculus~\cite{conf/lics/Padovani14} has been already mentioned in the context of deadlock-freedom for cyclic process networks.
His language is different from \APCP, as it lacks session constructs and does not have continuation-passing baked into the type system.
While it should be possible to encode sessions in his typed framework (using communication of pairs to model continuation-passing), it seems unclear how to transfer the analysis of deadlock-freedom from this setting to a language such as \APCP via such a translation.

\paragraph{Deadlock-freedom for cyclic process networks.}

We have already discussed the related works by Kobayashi~\cite{conf/concur/Kobayashi06}, Padovani~\cite{conf/lics/Padovani14}, and Dardha and Gay~\cite{conf/fossacs/DardhaG18}.
The work of Kobayashi and Laneve~\cite{journal/ic/KobayashiL17} is related to \APCP in that it addresses deadlock-freedom for \emph{unbounded} process networks.
 Toninho and Yoshida's work~\cite{journal/toplas/ToninhoY18} addresses deadlock-freedom for cyclic process networks by generating global types from binary types.
The work by Balzer, Toninho and Pfenning~\cite{conf/icfp/BalzerP17,conf/esop/BalzerTP19} is also worth mentioning: it guarantees deadlock-freedom for processes with shared, mutable resources by means of manifest sharing, i.e., explicitly acquiring and releasing access to resources.

\ifAPCP\else
\paragraph{Functional languages with sessions.}
We have already discussed the related works by Gay and Vasconcelos~\cite{journal/jfp/GayV10}, Wadler~\cite{conf/icfp/Wadler12}, Kokke and Dardha~\cite{conf/forte/KokkeD21,report/KokkeD21}, and Padovani and Novara~\cite{conf/forte/PadovaniN15}.
Lindley and Morris~\cite{conf/esop/LindleyM15} formally define a semantics for Wadler's \GV, prove deadlock-freedom and operational correspondence of the translation into Wadler's \CP, and give a translation from \CP into \GV.
Other related work in this line is Fowler \etal's \emph{Asynchronous} \GV with buffered asynchronous communication and a call-by-value semantics, and \emph{Exceptional} \GV (\EGV) which extends the former with exception handling~\cite{thesis/Fowler19,conf/popl/FowlerLMD19}.

\Cref{tbl:LASTn:compare} summarizes the comparison of \LASTn to some of these related works.
Note that none of the mentioned works supports sending anything but values, unlike \LASTn.
Our work ensures deadlock-freedom for configurations with cyclic topologies by means of a translation into a system with an established deadlock-freedom result (cf.\ \Cref{r:LASTn:LASTn:transDfComplex}); the alternative approach is to enhance \GV's type system with priorities, as done by Padovani and Novara~\cite{conf/forte/PadovaniN15} and Kokke and Dardha~\cite{conf/forte/KokkeD21,report/KokkeD21}.

\begin{table}[!t]
    \begin{center}
        \begin{tabular}{|c|c|c|c|c|c|}
            \hline
            & \LAST
            & \GV
            & \EGV
            & \PGV
            & \textbf{\LASTn}
            \\
            & \cite{journal/jfp/GayV10}
            & \!\cite{conf/icfp/Wadler12}\!
            & \!\cite{conf/popl/FowlerLMD19}\!
            & \!\cite{conf/forte/KokkeD21,report/KokkeD21}\!
            & \!\textbf{(this \ifthesis chapter\else paper\fi)}\!
            \\ \hline
            Communication
            & Async.
            & Sync.
            & Async.
            & Sync.
            & Async.
            \\ \hline
            \makecell{Cyclic \\ Topologies}
            & $\checkmark$
            & $\times$
            & $\times$
            & $\checkmark$
            & $\checkmark$
            \\ \hline
            \makecell{Deadlock- \\ freedom}
            & $\times$
            & \makecell{$\checkmark$ \\ (typing)}
            & \makecell{$\checkmark$ \\ (typing)}
            & \makecell{$\checkmark$ \\ (typing)}
            & \makecell{$\checkmark$ \\ (via \APCP)}
            \\ \hline
        \end{tabular}
    \end{center}
    \caption{The features of \LASTn compared to related works.}
    \label{tbl:LASTn:compare}
\end{table}

\paragraph{Typed encodings between (concurrent) \lamcalci and $\pi$-calculi.}
Several prior works develop typed encodings between $\lambda$- and $\pi$-calculi; to our knowledge, all of them consider call-by-value semantics, most do not consider $\lambda$ with message-passing, and none translate variables as sends, as we do.
Some of these works have been discussed  in \Cref{s:LASTn:LAST:trans}, where we justified the notions of completeness and soundness relevant for our translation of \LASTn into \APCP.
Vasconcelos~\cite{report/Vasconcelos00} translates an untyped \lamcalc into an input/output-typed $\pi$-calculus; the translation is sound, up to barbed congruence.
The aforementioned work by Wadler~\cite{conf/icfp/Wadler12,journal/jfp/Wadler14} translates the session-typed \GV with synchronous message-passing into the session-typed \CP, without giving a semantics for \GV nor associated completeness/soundness results.
Lindley and Morris~\cite{conf/esop/LindleyM15} formalize a semantics for \GV and its translation into~\CP, and prove completeness; no soundness result is given.
Kokke and Dardha~\cite{conf/forte/KokkeD21} define \PGV, an extension of \GV with priority annotations to support deadlock-free cyclic thread configurations; they give a translation from \PCP into \PGV that is both sound and complete, but no translation in the other direction is studied.
Toninho \etal~\cite{conf/fossacs/ToninhoCP12} translate a linearly typed \lamcalc into a session-typed $\pi$-calculus derived from intuitionistic linear logic; they prove completeness, but for soundness they extend reduction in $\lambda$ to call-by-name.
Toninho and Yoshida~\cite{journal/toplas/ToninhoY21} translate a linearly typed \lamcalc without message-passing into a session-typed polymorphic $\pi$-calculus. Their translation is fully abstract: it is complete and sound, with the latter result requiring an extension of reduction in $\pi$ with commuting conversions. 
\fi

\paragraph{Other applications of \APCP.}
\ifAPCP
In the introduction we mentioned one salient aspects of Curry-Howard correspondences for session types, namely analysis of deadlock-freedom.
In this \ifthesis chapter \else paper \fi we have addressed this aspect directly.
\else
As mentioned in the introduction, two salient aspects of Curry-Howard correspondences for session types, namely analysis of deadlock-freedom and connections with functional calculi, define the main themes of this \ifthesis chapter \else paper\fi.
\fi
Yet another highlight of the logical correspondences is their suitability for the analysis of \emph{multiparty protocols}.
In separate work~\cite{journal/scico/vdHeuvelP22}\ifthesis (cf.\ \Cref{c:APCP})\fi, we have devised a methodology for the analysis of multiparty session types (MPST) based on \APCP~\cite{journal/acm/HondaYC16}.
In the multiparty context, two or more participants interact following a common protocol.
\APCP is well-suited for the analysis of process implementations of MPST, which rely on asynchrony and recursion.
In our analysis, multiple separate processes implement the roles of one or more participants of a multiparty protocol.
The support for cyclic process networks in \APCP allows us to connect these implementations with each other directly, without requiring an additional process (as used in similar prior works~\cite{conf/forte/CairesP16,conf/concur/CarboneMSY15,conf/concur/CarboneLMSW16}).
The asynchrony in \APCP captures the asynchronous nature of MPST, and \APCP's recursion enables an expressive class of supported protocols.
This way, our methodology unlocks the transfer of (static) analysis of deadlock-freedom and protocol conformance from \APCP to distributed implementations of MPST. 
The key ideas underlying the methodology in~\cite{journal/scico/vdHeuvelP22} can be applied to the \emph{run-time} verification of such distributed implementations, as recently shown in~\cite{conf/rv/vdHeuvelPD23}.

\section{Closing Remarks}
\label{s:APCP:concl}

In this \ifthesis chapter\else paper\fi, we have presented two contributions to the challenging issue of ensuring deadlock-free, message-passing interactions in a session-typed setting.
Our \emph{first contribution} is \APCP: a typed framework for deadlock-free cyclic process networks with asynchronous communication and recursion.
The design of \APCP and its type system has a solid basis, with tight ties to logic (``propositions as sessions''): our syntax, semantics and type system harmoniously integrate insights separately presented in prior works for different typed variants of the $\pi$-calculus.

The step from synchronous communication (as in \PCP~\cite{conf/fossacs/DardhaG18}) to asynchronous communication (in \APCP) is significant.
In combination with cyclic process networks, asynchronous communication enables checking a larger class of deadlock-free processes: processes that under synchronous communication would be deadlocked due to blocking outputs may be deadlock-free in \APCP (see, e.g., \Cref{x:APCP:redd}).
Perhaps more significantly, asynchronous communication simplifies the priority management involved in the detection of cyclic dependencies (cf.\ \Cref{r:APCP:simplerPrios}).

\ifAPCP
We validated \APCP's design and expressiveness: as a classic example of modeling realistic concurrency patterns, we formalized Milner's  cyclic scheduler~\cite{book/Milner89} and analyzed its deadlock-freedom.
\else
\ifLASTn In this chapter, we \else In our \emph{second contribution}, we move to \fi consider cyclic process networks with asynchronous communication in the setting of a prototypical functional language.
We introduced \LASTn, a new functional language with asynchronous, session-typed communication.
Inspired by Gay and Vasconcelos' \LAST, the design of \LASTn combines buffered communication with a call-by-name reduction strategy, explicit substitutions, and explicit session closing. 
In our opinion, this design makes \LASTn an interesting object of study on its own.
In the spirit of Gay and Vasconcelos' work, we defined a deliberately simple type system for \LASTn, whose type preservation property ensures protocol fidelity and communication safety, but not deadlock-freedom.
This means that well-typed threads in \LASTn can perform session protocols that may lead to stuck terms/configurations. 
\ifLASTn To connect to the world of processes and to recover deadlock-freedom guarantees to the functional realm\else To connect our contributions, and to transfer the deadlock-freedom guarantee to the functional realm\fi, we presented a translation from \LASTn to \APCP \ifLASTn (a session-typed variant of the asynchronous \picalc) \fi that is \emph{operationally correct}.
In particular, operational correctness includes the \emph{soundness} property, which ensures that the translation reflects reduction steps and is critical to the transfer of deadlock-freedom formalized by \Cref{t:LASTn:LASTn:df}.
This way, we were able to analyze deadlock-freedom in \LASTn by leaning on the established results for \APCP (\ifLASTn\ifappendix\Cref{t:LASTn:APCP:df}\else{}cf.\ \cite{thesis/vdHeuvel24ex}\fi\else\Cref{t:APCP:df}\fi) without sacrificing complexity (cf.\ \Cref{r:LASTn:LASTn:transDfComplex}).
\fi

\ifAPCP\else
In future work, it would be interesting to study the extension of \LASTn with recursion/recursive types, and to extend our translation into \APCP  (and its operational correspondence results) accordingly.
We do not expect technical difficulties, in particular if \LASTn is enhanced with the kind of tail-recursive behavior present in \APCP.
\fi
\ifLASTn\else
\ifAPCP In future work, it \else In the same spirit, it \fi would be insightful to explore how to accommodate \APCP's support for priorities into a process language with inductive and coinductive types (least and greatest fixed points, respectively), such as the one studied by Rocha and Caires~\cite{conf/esop/RochaC23}.
\fi

\paragraph{Acknowledgements}
We are grateful to the anonymous reviewers for their careful reading of our \ifthesis chapter\else paper\fi and their useful feedback.
We also thank Ornela Dardha for clarifying the typing rules of \PCP to us, and Simon Fowler for his helpful feedback.
We gratefully acknowledge the support of the Dutch Research Council (NWO) under project No.\,016.Vidi.189.046 (Unifying Correctness for Communicating Software).

\bibliographystyle{alphaurl}
\bibliography{refs}

\appendix
\newpage

\tableofcontents

\newpage
\section{\protect\LASTn: Detailed Definitions and Proofs}
\label{as:LASTn}

\subsection{Self-contained Definition of \protect\LASTn and its Type System}
\label{as:LASTn:lang}


\begin{figure}[p]
    Terms ($\term{M},\term{N},\term{L}$) and reduction contexts ($\tCtx{R}$):
    \begin{align*}
        \term{M},\term{N},\term{L}
        & \begin{array}[t]{@{}l@{\kern.5ex}lr@{\kern1.7ex}l@{\kern.5ex}lr@{}}
            {} ::= &
            \term{x} & \text{variable}
            & \sepr &
            \term{\tNew} & \text{create new channel}
            \\ \sepr* &
            \term{()} & \text{unit value}
            & \sepr &
            \term{\tSpawn M ; N} & \text{spawn $\term{M}$ in parallel to $\term{N}$}
            \\ \sepr* &
            \term{\lam x . M} & \text{abstraction}
            & \sepr &
            \term{(M,N)} & \text{pair construction}
            \\ \sepr* &
            \term{M~N} & \text{application}
            & \sepr &
            \term{\tLet (x,y)=M \tIn N} & \text{pair deconstruction}
            \\ \sepr* &
            \term{\tSend M \, N} & \text{send $\term{M}$ along $\term{N}$}
            & \sepr &
            \term{\tSel \ell \, M} & \text{select label $\ell$ along $\term{M}$}
            \\ \sepr* &
            \term{\tRecv M} & \text{receive along $\term{M}$}
            & \sepr &
            \term{\tCase M \tOf \{ i : M \}_{i \in I}} & \text{offer labels in $I$ along $\term{M}$}
            \\ \sepr* &
            \term{\tClose M ; N} & \text{close $\term{M}$}
            & \sepr &
            \term{M \tSub{ N/x }} & \text{explicit substitution}
        \end{array}
        \\
        \term{\tCtx{R}}
        & \begin{array}[t]{@{}l@{\kern.5ex}l@{}}
            {} ::= &
            \term{\tHole}
            \sepr
            \term{\tCtx{R}~M}
            \sepr
            \term{\tSend M \, \tCtx{R}}
            \sepr
            \term{\tRecv \tCtx{R}}
            \sepr
            \term{\tLet (x,y) = \tCtx{R} \tIn M}
            \\ \sepr* &
            \term{\tSel \ell\, \tCtx{R}}
            \sepr
            \term{\tCase \tCtx{R} \tOf \{ i : M \}_{i \in I}}
            \sepr
            \term{\tClose \tCtx{R} ; M}
            \sepr
            \term{\tCtx{R} \tSub{ M/x }}
        \end{array}
    \end{align*}

    \dashes

    Structural congruence for terms ($\equivM$) and term reduction ($\reddM$):
    \begin{mathpar}
        \begin{bussproof}[sc-sub-ext]
            \bussAssume{
                \term{x} \notin \fv(\tCtx{R})
            }
            \bussUn{
                \term{(\tCtx{R}[M]) \tSub{ N/x }}
                \equivM
                \term{\tCtx*{R}[M \tSub{ N/x }]}
            }
        \end{bussproof}
        \and
        \begin{bussproof}[red-lam]
            \bussAx{
                \term{(\lam x . M)~N}
                \reddM
                \term{M \tSub{ N/x }}
            }
        \end{bussproof}
        \and
        \begin{bussproof}[red-pair]
            \bussAx{
                \term{\tLet (x,y) = (M_1,M_2) \tIn N}
                \reddM
                \term{N \tSub{ M_1/x,M_2/y }}
            }
        \end{bussproof}
        \and
        \begin{bussproof}[red-name-sub]
            \bussAx{
                \term{x \tSub{ M/x }}
                \reddM
                \term{M}
            }
        \end{bussproof}
        \and
        \begin{bussproof}[red-lift]
            \bussAssume{
                \term{M}
                \reddM
                \term{N}
            }
            \bussUn{
                \tCtx{R}[M]
                \reddM
                \tCtx{R}[N]
            }
        \end{bussproof}
        \and
        \begin{bussproof}[red-lift-sc]
            \bussAssume{
                \term{M}
                \equivM
                \term{M'}
            }
            \bussAssume{
                \term{M'}
                \reddM
                \term{N'}
            }
            \bussAssume{
                \term{N'}
                \equivM
                \term{N}
            }
            \bussTern{
                \term{M}
                \reddM
                \term{N}
            }
        \end{bussproof}
    \end{mathpar}

    \caption{The \protect\LASTn term language.}\label{f:LASTn:LASTn:terms}
\end{figure}

\begin{figure}[p]
    \def\MathparLineskip{\lineskip=3pt}
    Markers ($\term{\phi}$), messages ($\term{m},\term{n}$), configurations ($\term{C},\term{D},\term{E}$), thread ($\tCtx{F}$) and configuration ($\tCtx{G}$) contexts:
    
    \vspace{-2.9em}
    \begin{align*}
        \term{\phi} ::= {}
        & \term{\tMain} \sepr \term{\tChild}
        &
        \term{m},\term{n} ::= {}
        & \term{M} \sepr \term{\ell}
        \\
        \term{C},\term{D},\term{E} ::= {}
        & \term{\phi\, M} \sepr \term{C \prl D} \sepr \term{\pRes{x\tBfr{\vec{m}}y} C} \sepr \term{C \tSub{ M/x }}
        \\
        \tCtx{F} ::= {}
        & \term{\phi\, \tCtx{R}}
        &
        \term{\tCtx{G}} ::= {}
        & \term{\tHole} \sepr \term{\tCtx{G} \prl C} \sepr \term{\pRes{x\tBfr{\vec{m}}y} \tCtx{G}} \sepr \term{\tCtx{G} \tSub{ M/x }}
    \end{align*}

    \vspace{-1.2em}
    \dashes
    \vspace{-.7em}

    Structural congruence for configurations ($\equivC$):

    \vspace{-5ex}
    \begin{mathpar}
        \begin{bussproof}[sc-term-sc]
            \bussAssume{
                \term{M}
                \equivM
                \term{M'}
            }
            \bussUn{
                \term{\phi\, M}
                \equivC
                \term{\phi\, M'}
            }
        \end{bussproof}
        \and
        \begin{bussproof}[sc-res-swap]
            \bussAx{
                \term{\pRes{x\tBfr{\epsi}y} C}
                \equivC
                \term{\pRes{y\tBfr{\epsi}x} C}
            }
        \end{bussproof}
        \and
        \begin{bussproof}[sc-res-comm]
            \bussAx{
                \term{\pRes{x\tBfr{\vec{m}}y} \pRes{z\tBfr{\vec{n}}w} C}
                \equivC
                \term{\pRes{z\tBfr{\vec{n}}w} \pRes{x\tBfr{\vec{m}}y} C}
            }
        \end{bussproof}
        \and
        \begin{bussproof}[sc-res-ext]
            \bussAssume{
                \term{x},\term{y} \notin \fv(\term{C})
            }
            \bussUn{
                \term{\pRes{x\tBfr{\vec{m}}y} ( C \prl D )}
                \equivC
                \term{C \prl \pRes{x\tBfr{\vec{m}}y} D}
            }
        \end{bussproof}
        \and
        \begin{bussproof}[sc-par-comm]
            \bussAx{
                \term{C \prl D}
                \equivC
                \term{D \prl C}
            }
        \end{bussproof}
        \and
        \begin{bussproof}[sc-par-assoc]
            \bussAx{
                \term{C \prl (D \prl E)}
                \equivC
                \term{(C \prl D) \prl E}
            }
        \end{bussproof}
        \and
        \begin{bussproof}[sc-conf-sub]
            \bussAx{
                \term{\phi\,(M \tSub{ N/x })}
                \equivC
                \term{(\phi\,M) \tSub{ N/x }}
            }
        \end{bussproof}
        \and
        \begin{bussproof}[sc-conf-sub-ext]
            \bussAssume{
                \term{x} \notin \fv(\tCtx{G})
            }
            \bussUn{
                \term{(\tCtx{G}[C]) \tSub{ M/x }}
                \equivC
                \term{\tCtx{G}[C \tSub{ M/x }]}
            }
        \end{bussproof}
    \end{mathpar}

    \vspace{-1.2em}
    \dashes
    \vspace{-.7em}

    Configuration reduction ($\reddC$):

    \vspace{-5ex}
    \begin{mathpar}
        \begin{bussproof}[red-new]
            \bussAx{
                \tCtx{F}[\tNew]
                \reddC
                \term{\pRes{x\tBfr{\epsi}y} ( \tCtx{F}[(x,y)] )}
            }
        \end{bussproof}
        \and
        \begin{bussproof}[red-spawn]
            \bussAx{
                \tCtx{\hat{F}}[\tSpawn M ; N]
                \reddC
                \term{\tCtx{\hat{F}}[N] \prl \tChild\,M}
            }
        \end{bussproof}
        \and
        \begin{bussproof}[red-send]
            \bussAx{
                \term{\pRes{x\tBfr{\vec{m}}y} ( \tCtx{\hat{F}}[\tSend M \, x] \prl C )}
                \reddC
                \term{\pRes{x\tBfr{M,\vec{m}}y} ( \tCtx{\hat{F}}[x] \prl C )}
            }
        \end{bussproof}
        \and
        \begin{bussproof}[red-recv]
            \bussAx{
                \term{\pRes{x\tBfr{\vec{m},M}y} ( \tCtx{\hat{F}}[\tRecv y] \prl C )}
                \reddC
                \term{\pRes{x\tBfr{\vec{m}}y} ( \tCtx{\hat{F}}[(M,y)] \prl C )}
            }
        \end{bussproof}
        \and
        \begin{bussproof}[red-select]
            \bussAx{
                \term{\pRes{x\tBfr{\vec{m}}y} ( \tCtx{F}[\tSel \ell\, x] \prl C )}
                \reddC
                \term{\pRes{x\tBfr{\ell,\vec{m}}y} ( \tCtx{F}[x] \prl C )}
            }
        \end{bussproof}
        \and
        \begin{bussproof}[red-case]
            \bussAssume{
                j \in I
            }
            \bussUn{
                \term{\pRes{x\tBfr{\vec{m},j}y} ( \tCtx{F}[\tCase y \tOf \{ i : M_i \}_{i \in I}] \prl C )}
                \reddC
                \term{\pRes{x\tBfr{\vec{m}}y} ( \tCtx{F}[M_j~y] \prl C )}
            }
        \end{bussproof}
        \and
        \begin{bussproof}[red-close]
            \bussAx{
                \term{\pRes{x\tBfr{\vec{m}}y} ( \tCtx{F}[\tClose x ; M] \prl C )}
                \reddC
                \term{\pRes{\tNil\tBfr{\vec{m}}y} ( \tCtx{F}[M] \prl C )}
            }
        \end{bussproof}
        \and
        \begin{bussproof}[red-res-nil]
            \bussAx{
                \term{\pRes{\tNil\tBfr{\epsilon}\tNil} C}
                \reddC
                \term{C}
            }
        \end{bussproof}
        \and
        \begin{bussproof}[red-par-nil]
            \bussAx{
                \term{C \prl \tChild\, ()}
                \reddC
                \term{C}
            }
        \end{bussproof}
        \and
        \begin{bussproof}[red-lift-C]
            \bussAssume{
                \term{C}
                \reddC
                \term{C'}
            }
            \bussUn{
                \tCtx{G}[C]
                \reddC
                \tCtx{G}[C']
            }
        \end{bussproof}
        \and
        \begin{bussproof}[red-lift-M]
            \bussAssume{
                \term{M}
                \reddM
                \term{M'}
            }
            \bussUn{
                \tCtx{F}[M]
                \reddC
                \tCtx{F}[M']
            }
        \end{bussproof}
        \and
        \begin{bussproof}[red-conf-lift-sc]
            \bussAssume{
                \term{C}
                \equivC
                \term{C'}
            }
            \bussAssume{
                \term{C'}
                \reddC
                \term{D'}
            }
            \bussAssume{
                \term{D'}
                \equivC
                \term{D}
            }
            \bussTern{
                \term{C}
                \reddC
                \term{D}
            }
        \end{bussproof}
    \end{mathpar}
    
    \vspace{-2.1em}
    \caption{The \protect\LASTn configuration language.}\label{f:LASTn:LASTn:confs}
\end{figure}

\begin{figure}[p]
    \def\defaultHypSeparation{\hskip1ex}
    \begin{mathpar}
        \begin{bussproof}[typ-var]
            \bussAx{
                \term{x}:\type{T} \vdashM \term{x}:\type{T}
            }
        \end{bussproof}
        \and
        \begin{bussproof}[typ-abs]
            \bussAssume{
                \type{\Gamma}, \term{x}:\type{T} \vdashM \term{M} :\type{U}
            }
            \bussUn{
                \type{\Gamma} \vdashM \term{\lam x . M} : \type{T \lolli U}
            }
        \end{bussproof}
        \and
        \begin{bussproof}[typ-app]
            \bussAssume{
                \type{\Gamma} \vdashM \term{M}: \type{T \lolli U}
            }
            \bussAssume{
                \type{\Delta} \vdashM \term{N}: \type{T}
            }
            \bussBin{
                \type{\Gamma}, \type{\Delta} \vdashM \term{M~N}: \type{U}
            }
        \end{bussproof}
        \and
        \begin{bussproof}[typ-unit]
            \bussAx{
                \type{\emptyset} \vdashM \term{()}: \type{\1}
            }
        \end{bussproof}
        \and
        \begin{bussproof}[typ-pair]
            \bussAssume{
                \type{\Gamma} \vdashM \term{M}: \type{T}
            }
            \bussAssume{
                \type{\Delta} \vdashM \term{N}: \type{U}
            }
            \bussBin{
                \type{\Gamma}, \type{\Delta} \vdashM \term{(M,N)}: \type{T \times U}
            }
        \end{bussproof}
        \and
        \begin{bussproof}[typ-split]
            \bussAssume{
                \type{\Gamma} \vdashM \term{M}: \type{T \times T'}
            }
            \bussAssume{
                \type{\Delta}, \term{x}:\type{T}, \term{y}:\type{T'} \vdashM \term{N}: \type{U}
            }
            \bussBin{
                \type{\Gamma}, \type{\Delta} \vdashM \term{\tLet (x,y) = M \tIn N}: \type{U}
            }
        \end{bussproof}
        \and
        \begin{bussproof}[typ-new]
            \bussAx{
                \type{\emptyset} \vdashM \term{\tNew}: \type{S \times \ol{S}}
            }
        \end{bussproof}
        \and
        \begin{bussproof}[typ-spawn]
            \bussAssume{
                \type{\Gamma} \vdashM \term{M}: \type{\1}
            }
            \bussAssume{
                \type{\Delta} \vdashM \term{N} : \type{T}
            }
            \bussBin{
                \type{\Gamma} , \type{\Delta} \vdashM \term{\tSpawn M ; N}: \type{T}
            }
        \end{bussproof}
        \and
        \begin{bussproof}[typ-close]
            \bussAssume{
                \type{\Gamma} \vdashM \term{M} : \type{\tEnd}
            }
            \bussAssume{
                \type{\Delta} \vdashM \term{N} : \type{T}
            }
            \bussBin{
                \type{\Gamma} , \type{\Delta} \vdashM \term{\tClose M ; N} : \type{T}
            }
        \end{bussproof}
        \and
        \begin{bussproof}[typ-send]
            \bussAssume{
                \type{\Gamma} \vdashM \term{M}: \type{T}
            }
            \bussAssume{
                \type{\Delta} \vdashM \term{N}: \type{{!}T . S}
            }
            \bussBin{
                \type{\Gamma} , \type{\Delta} \vdashM \term{\tSend M \, N}: \type{S}
            }
        \end{bussproof}
        \and
        \begin{bussproof}[typ-recv]
            \bussAssume{
                \type{\Gamma} \vdashM \term{M}: \type{{?}T . S}
            }
            \bussUn{
                \type{\Gamma} \vdashM \term{\tRecv M}: \type{T \times S}
            }
        \end{bussproof}
        \and
        \begin{bussproof}[typ-sel]
            \bussAssume{
                \type{\Gamma} \vdashM \term{M}: \type{\oplus \{ i : S_i \}_{i \in I}}
            }
            \bussAssume{
                j \in I
            }
            \bussBin{
                \type{\Gamma} \vdashM \term{\tSel j\, M}: \type{S_j}
            }
        \end{bussproof}
        \and
        \begin{bussproof}[typ-case]
            \bussAssume{
                \type{\Gamma} \vdashM \term{M}: \type{{\&}\{ i : S_i \}_{i \in I}}
            }
            \bussAssume{
                \forall i \in I.~ \type{\Delta} \vdashM \term{N_i}: \type{S_i \lolli U}
            }
            \bussBin{
                \type{\Gamma}, \type{\Delta} \vdashM \term{\tCase M \tOf \{ i : N_i \}_{i \in I}}: \type{U}
            }
        \end{bussproof}
        \and
        \begin{bussproof}[typ-sub]
            \bussAssume{
                \type{\Gamma}, \term{x}:\type{T} \vdashM \term{M}: \type{U}
            }
            \bussAssume{
                \type{\Delta} \vdashM \term{N}: \type{T}
            }
            \bussBin{
                \type{\Gamma}, \type{\Delta} \vdashM \term{M \tSub{ N/x }}: \type{U}
            }
        \end{bussproof}
        \\ \dashes \\
        \begin{bussproof}[typ-buf]
            \bussAx{
                \type{\emptyset} \vdashB \term{\tBfr{\epsilon}}: \type{S'} > \type{S'}
            }
        \end{bussproof}
        \and
        \begin{bussproof}[typ-buf-send]
            \bussAssume{
                \type{\Gamma} \vdashM \term{M}: \type{T}
            }
            \bussAssume{
                \type{\Delta} \vdashB \term{\tBfr{\vec{m}}}: \type{S'} > \type{S}
            }
            \bussBin{
                \type{\Gamma}, \type{\Delta} \vdashB \term{\tBfr{\vec{m},M}}: \type{S'} > \type{{!}T . S}
            }
        \end{bussproof}
        \and
        \begin{bussproof}[typ-buf-sel]
            \bussAssume{
                \type{\Gamma} \vdashB \term{\tBfr{\vec{m}}}: \type{S'} > \type{S_j}
            }
            \bussAssume{
                j \in I
            }
            \bussBin{
                \type{\Gamma} \vdashB \term{\tBfr{\vec{m},j}}: \type{S'} > \type{\oplus \{ i : S_i \}_{i \in I}}
            }
        \end{bussproof}
        \and
        \begin{bussproof}[typ-buf-end-L]
            \bussAx{
                \type{\emptyset} \vdashB \term{\tBfr{\epsi}} : \type{\tEnd}>\type{\tNil}
            }
        \end{bussproof}
        \and
        \begin{bussproof}[typ-buf-end-R]
            \bussAx{
                \type{\emptyset} \vdashB \term{\tBfr{\epsi}} : \type{\tNil}>\type{\tEnd}
            }
        \end{bussproof}
        \\ \dashes \\
        \begin{bussproof}[typ-main]
            \bussAssume{
                \type{\Gamma} \vdashM \term{M}: \type{\hat{T}}
            }
            \bussUn{
                \type{\Gamma} \vdashC{\tMain} \term{\tMain\, M}: \type{\hat{T}}
            }
        \end{bussproof}
        \and
        \begin{bussproof}[typ-child]
            \bussAssume{
                \type{\Gamma} \vdashM \term{M}: \type{\1}
            }
            \bussUn{
                \type{\Gamma} \vdashC{\tChild} \term{\tChild\, M}: \type{\1}
            }
        \end{bussproof}
        \and
        \begin{bussproof}[typ-par]
            \bussAssume{
                \type{\Gamma} \vdashC{\phi_1} \term{C}: \type{T_1}
            }
            \bussAssume{
                \type{\Delta} \vdashC{\phi_2} \term{D}: \type{T_2}
            }
            \bussAssume{
                \{\type{T_1},\type{T_2}\} \setminus \{\type{T}\} = \type{\1}
            }
            \bussTern{
                \type{\Gamma}, \type{\Delta} \vdashC{\phi_1 + \phi_2} \term{C \prl D}: \type{T}
            }
        \end{bussproof}
        \and
        \begin{bussproof}[typ-res]
            \bussAssume{
                \type{\Gamma} \vdashB \term{\tBfr{\vec{m}}}: \type{S'} > \type{S}
            }
            \bussAssume{
                \type{\Delta}, \term{x}: \type{S'} \vdashC{\phi} \term{C}: \type{T}
            }
            \bussAssume{
                \type{\Gamma'} , \term{y}:\type{\ol{S}} = \type{\Gamma} , \type{\Delta}
            }
            \bussTern{
                \type{\Gamma'} \vdashC{\phi} \term{\pRes{x\tBfr{\vec{m}}y} C}: \type{T}
            }
        \end{bussproof}
        \and
        \begin{bussproof}[typ-conf-sub]
            \bussAssume{
                \type{\Gamma}, \term{x}:\type{T} \vdashC{\phi} \term{C}: \type{U}
            }
            \bussAssume{
                \type{\Delta} \vdashM \term{M}: \type{T}
            }
            \bussBin{
                \type{\Gamma}, \type{\Delta} \vdashC{\phi} \term{C \tSub{ M/x }}: \type{U}
            }
        \end{bussproof}
    \end{mathpar}
    \caption{\protect\LASTn typing rules for terms (top), buffers (center), and configurations (bottom).}
    \label{f:LASTn:LASTn:type}
\end{figure}

\Cref{f:LASTn:LASTn:terms} gives the term language of \LASTn, \Cref{f:LASTn:LASTn:confs} the configuration language, and \Cref{f:LASTn:LASTn:type} the type system.


\newpage
\subsection{Type Preservation}
\label{as:LASTn:tp}


Here, we prove type preservation for \LASTn:

\tLASTnTp*

The proof is split into two parts: subject congruence and subject reduction.
These and intermediate results are organized as follows:
\begin{itemize}

    \item
        \Cref{t:LASTn:LASTn:scTerms,t:LASTn:LASTn:srTerms} prove subject congruence and subject reduction for terms, respectively.

    \item
        \Cref{t:LASTn:LASTn:sc,t:LASTn:LASTn:sr} then prove subject congruence and subject reduction for configurations, respectively, from which \Cref{t:LASTn:LASTn:tp} follows.

\end{itemize}

\begin{theorem}[Subject Congruence for Terms]
\label{t:LASTn:LASTn:scTerms}
    If $\type{\Gamma} \vdashM \term{M} : \type{T}$ and $\term{M} \equivM \term{N}$, then $\type{\Gamma} \vdashM \term{N} : \type{T}$.
\end{theorem}

\begin{proof}
    By induction on the derivation of $\term{M} \equivM \term{N}$.
    The inductive cases follow from the IH directly.
    We consider the only \proofScTermsSubExt
\end{proof}

\begin{theorem}[Subject Reduction for Terms]
\label{t:LASTn:LASTn:srTerms}
    If $\type{\Gamma} \vdashM \term{M} : \type{T}$ and $\term{M} \reddM \term{N}$, then $\type{\Gamma} \vdashM \term{N} : \type{T}$.
\end{theorem}

\begin{proof}
    By induction on the derivation of $\term{M} \reddM \term{N}$ (\ih{1}).
    The case of Rule~\ruleLabel{red-lift} follows by induction on the structure of the reduction context $\tCtx{R}$, where the base case ($\tCtx{R} = \tHole$) follows from \ih{1}.
    The case of Rule~\ruleLabel{red-lift-sc} follows from \ih{1} and \Cref{t:LASTn:LASTn:scTerms} (subject congruence for terms).
    We consider the other cases, applying inversion of typing and deriving the typing of the term after reduction:
    \begin{itemize}

        \item
            Rule~\ruleLabel{red-lam}: $\term{(\lam x . M)~N} \reddM \term{M \tSub{ N/x }}$.
            \begin{align*}
                & \begin{bussproof}
                    \bussAssume{
                        \type{\Gamma}, \term{x}:\type{T} \vdashM \term{M}: \type{U}
                    }
                    \bussUn[\ruleLabel{typ-abs}]{
                        \type{\Gamma} \vdashM \term{\lam x . M}: \type{T \lolli U}
                    }
                    \bussAssume{
                        \type{\Delta} \vdashM \term{N}: \type{T}
                    }
                    \bussBin[\ruleLabel{typ-app}]{
                        \type{\Gamma}, \type{\Delta} \vdashM \term{(\lam x . M)~N}: \type{U}
                    }
                \end{bussproof}
                \\
                & \reddM
                \\
                & \begin{bussproof}
                    \bussAssume{
                        \type{\Gamma}, \term{x}:\type{T} \vdashM \term{M}: \type{U}
                    }
                    \bussAssume{
                        \type{\Delta} \vdashM \term{N}: \type{T}
                    }
                    \bussBin[\ruleLabel{typ-sub}]{
                        \type{\Gamma}, \type{\Delta} \vdashM \term{M \tSub{ N/x }}: \type{U}
                    }
                \end{bussproof}
            \end{align*}

        \item
            Rule~\ruleLabel{red-pair}: $\term{\tLet (x,y) = (M_1,M_2) \tIn N} \reddM \term{N \tSub{ M_1/x,M_2/y }}$.
            \begin{align*}
                & \begin{bussproof}
                    \bussAssume{
                        \type{\Gamma} \vdashM \term{M_1}: \type{T}
                    }
                    \bussAssume{
                        \type{\Gamma'} \vdashM \term{M_2}: \type{T'}
                    }
                    \bussBin[\ruleLabel{typ-pair}]{
                        \type{\Gamma}, \type{\Gamma'} \vdashM \term{(M_1,M_2)}: \type{T \times T'}
                    }
                    \bussAssume{
                        \type{\Delta}, \term{x}:\type{T}, \term{y}:\type{T'} \vdashM \term{N}: \type{U}
                    }
                    \bussBin[\ruleLabel{typ-split}]{
                        \type{\Gamma}, \type{\Gamma'}, \type{\Delta} \vdashM \term{\tLet (x,y) = (M_1,M_2) \tIn N}: \type{U}
                    }
                \end{bussproof}
                \\
                & \reddM
                \\
                & \begin{bussproof}
                    \bussAssume{
                        \type{\Delta}, \term{x}:\type{T}, \term{y}:\type{T'} \vdashM \term{N}: \type{U}
                    }
                    \bussAssume{
                        \type{\Gamma} \vdashM \term{M_1}: \type{T}
                    }
                    \bussBin[\ruleLabel{typ-sub}]{
                        \type{\Gamma}, \type{\Delta}, \term{y}:\type{T'} \vdashM \term{N \tSub{ M_1/x }}: \type{U}
                    }
                    \bussAssume{
                        \type{\Gamma'} \vdashM \term{M_2}: \type{T'}
                    }
                    \bussBin[\ruleLabel{typ-sub}]{
                        \type{\Gamma}, \type{\Gamma'}, \type{\Delta} \vdashM \term{N \tSub{ M_1/x,M_2/y }}: \type{U}
                    }
                \end{bussproof}
            \end{align*}

        \item
            Rule~\ruleLabel{red-name-sub}: $\term{x \tSub{ M/x }} \reddM \term{M}$.
            \begin{align*}
                & \begin{bussproof}
                    \bussAx[\ruleLabel{typ-var}]{
                        \term{x}:\type{U} \vdashM \term{x}: \type{U}
                    }
                    \bussAssume{
                        \type{\Gamma} \vdashM \term{M}: \type{U}
                    }
                    \bussBin[\ruleLabel{typ-sub}]{
                        \type{\Gamma} \vdashM \term{x \tSub{ M/x }}: \type{U}
                    }
                \end{bussproof}
                \\
                & \reddM
                \\
                & \type{\Gamma} \vdashM \term{M}: \type{U}
                \tag*{\qedhere}
            \end{align*}
    \end{itemize}
\end{proof}

\begin{theorem}[Subject Congruence for Configurations]
    \label{t:LASTn:LASTn:sc}
    If $\type{\Gamma} \vdashC{\phi} \term{C}: \type{T}$ and $\term{C} \equivC \term{D}$, then $\type{\Gamma} \vdashC{\phi} \term{D} : \type{T}$.
\end{theorem}

\begin{proof}
    By induction on the derivation of $\term{C} \equivC \term{D}$.
    The inductive cases follow from the IH directly.
    The case for Rule~\ruleLabel{sc-term-sc} follows from \Cref{t:LASTn:LASTn:scTerms} (subject congruence for terms).
    The cases for Rules~\ruleLabel{sc-res-comm}, \ruleLabel{sc-par-nil}, \ruleLabel{sc-par-comm}, and~\ruleLabel{par-assoc} are straightforward.
    We consider the other cases:
    \begin{itemize}
        \item
            \proofScResSwap

        \item
            Rule~\ruleLabel{sc-res-ext}: $\term{x},\term{y} \notin \fv(\term{C}) \implies \term{\pRes{x\tBfr{\vec{m}}y}(C \prl D)} \equivC \term{C \prl \pRes{x\tBfr{\vec{m}}y}D}$.

            The analysis depends on whether $\term{C}$ or $\term{D}$ are child threads.
            W.l.o.g., we assume $\term{C}$ is a child thread.
            Assuming $\term{x},\term{y} \notin \fn(\term{C})$, we apply inversion of typing:
            \[
                \begin{bussproof}
                    \bussAssume{
                        \type{\Gamma} \vdashB \term{\tBfr{\vec{m}}}: \type{S'} > \type{S}
                    }
                    \bussAssume{
                        \type{\Delta} \vdashC{\tChild} \term{C}: \type{\1}
                    }
                    \bussAssume{
                        \type{\Lambda}, \term{x}:\type{S'}, \term{y}:\type{\ol{S}} \vdashC{\phi} \term{D}: \type{T}
                    }
                    \bussBin[\ruleLabel{typ-par-L}]{
                        \type{\Delta}, \type{\Lambda}, \term{x}:\type{S'}, \term{y}:\type{\ol{S}} \vdashC{\tChild+\phi} \term{C \prl D}: \type{T}
                    }
                    \bussBin[\ruleLabel{typ-res}]{
                        \type{\Gamma}, \type{\Delta}, \type{\Lambda} \vdashC{\tChild+\phi} \term{\pRes{x\tBfr{\vec{m}}y}(C \prl D)}: \type{T}
                    }
                \end{bussproof}
            \]
            Then, we derive the typing of the structurally congruent configuration:
            \[
                \begin{bussproof}
                    \bussAssume{
                        \type{\Delta} \vdashC{\tChild} \term{C}: \type{\1}
                    }
                    \bussAssume{
                        \type{\Gamma} \vdashB \term{\tBfr{\vec{m}}}: \type{S'} > \type{S}
                    }
                    \bussAssume{
                        \type{\Lambda}, \term{x}:\type{S'}, \term{y}:\type{\ol{S}} \vdashC{\phi} \term{D}: \type{T}
                    }
                    \bussBin[\ruleLabel{typ-res}]{
                        \type{\Gamma}, \type{\Lambda} \vdashC{\phi} \term{\pRes{x\tBfr{\vec{m}}y}D}: \type{T}
                    }
                    \bussBin[\ruleLabel{typ-par-L}]{
                        \type{\Gamma}, \type{\Delta}, \type{\Lambda} \vdashC{\tChild+\phi} \term{C \prl \pRes{x\tBfr{\vec{m}}y}D}: \type{T}
                    }
                \end{bussproof}
            \]
            The other direction is analogous.

        \item
            Rule~\ruleLabel{sc-conf-sub}: $\term{\phi\,(M \tSub{ N/x })} \equivC \term{(\phi\,M) \tSub{ N/x }}$.

            This case follows by a straightforward inversion of typing on both terms:
            \begin{align*}
                & \begin{bussproof}
                    \bussAssume{
                        \type{\Gamma}, \term{x}:\type{T} \vdashM \term{M}: \type{U}
                    }
                    \bussAssume{
                        \type{\Delta} \vdashM \term{N}: \type{T}
                    }
                    \bussBin[\ruleLabel{typ-sub}]{
                        \type{\Gamma}, \type{\Delta} \vdashM \term{M \tSub{ N/x }}: \type{U}
                    }
                    \bussUn[\ruleLabel{typ-main/child}]{
                        \type{\Gamma}, \type{\Delta} \vdashC{\phi} \term{\phi\,(M \tSub{ N/x })}: \type{U}
                    }
                \end{bussproof}
                \\
                & \equivC
                \\
                & \begin{bussproof}
                    \bussAssume{
                        \type{\Gamma}, \term{x}:\type{T} \vdashM \term{M}: \type{U}
                    }
                    \bussUn[\ruleLabel{typ-main/child}]{
                        \type{\Gamma}, \term{x}:\type{T} \vdashC{\phi} \term{\phi\,M}: \type{U}
                    }
                    \bussAssume{
                        \type{\Delta} \vdashM \term{N}: \type{T}
                    }
                    \bussBin[\ruleLabel{typ-conf-sub}]{
                        \type{\Gamma}, \type{\Delta} \vdashC{\phi} \term{(\phi\,M) \tSub{ N/x }}: \type{U}
                    }
                \end{bussproof}
            \end{align*}

        \item
            Rule~\ruleLabel{sc-conf-sub-ext}: $\term{x} \notin \fv(\tCtx{G}) \implies \term{(\tCtx{G}[C]) \tSub{ M/x }} \equivC \term{\tCtx{G}[C \tSub{ M/x }]}$.

            This case follows by induction on the structure of $\tCtx{G}$.
            The inductive cases follow from the IH straightforwardly.
            For the base case ($\tCtx{G} = \tHole$), the structural congruence is simply an equality.
            \qedhere

    \end{itemize}
\end{proof}

\begin{theorem}[Subject Reduction for Configurations]
\label{t:LASTn:LASTn:sr}
    If $\type{\Gamma} \vdashC{\phi} \term{C}: \type{T}$ and $\term{C} \reddC \term{D}$, then $\type{\Gamma} \vdashC{\phi} \term{D}: \type{T}$.
\end{theorem}

\begin{proof}
    By induction on the derivation of $\term{C} \reddC \term{D}$ (\ih{1}).
    The case of Rule~\ruleLabel{red-lift-C} \mbox{($\term{C} \reddC \term{C'} \implies \tCtx{G}[C] \reddC \tCtx{G}[C']$)} follows by induction on the structure of $\tCtx{G}$, directly from \ih{1}.
    The case of Rule~\ruleLabel{red-lift-M} ($\term{M} \reddM \term{M'} \implies \tCtx{F}[M] \reddC \tCtx{F}[M']$) follows by induction on the structure of $\tCtx{F}$, where the base case ($\tCtx{F} = \term{\phi\,\tHole}$) follows from \Cref{t:LASTn:LASTn:srTerms} (subject reduction for terms).
    The case for Rule~\ruleLabel{red-conf-lift-sc} ($\term{C} \equivC \term{C'} \wedge \term{C'} \reddC \term{D'} \wedge \term{D'} \equivC \term{D} \implies \term{C} \reddC \term{D}$) follows from \ih{1} and \Cref{t:LASTn:LASTn:sc} (subject congruence for configurations).
    We consider the other cases:
    \begin{itemize}

        \item
            Rule~\ruleLabel{red-new}: $\tCtx{F}[\tNew] \reddC \term{\pRes{x\tBfr{\epsilon}y} ( \tCtx{F}[(x,y)] )}$.

            This case follows by induction on the structure of $\tCtx{F}$.
            The inductive cases follow from the IH directly.
            For the base case ($\tCtx{F} = \term{\phi\,\tHole}$), we apply inversion of typing and derive the typing of the reduced configuration:
            \begin{align*}
                & \begin{bussproof}
                    \bussAx[\ruleLabel{typ-new}]{
                        \type{\emptyset} \vdashM \term{\tNew}: \type{S \times \ol{S}}
                    }
                    \bussUn[\ruleLabel{typ-main/child}]{
                        \type{\emptyset} \vdashC{\phi} \term{\phi\,\tNew}: \type{S \times \ol{S}}
                    }
                \end{bussproof}
                \\
                & \reddC
                \\
                & \begin{bussproof}
                    \bussAx[\ruleLabel{typ-buf}]{
                        \type{\emptyset} \vdashB \term{\tBfr{\epsi}}: \type{S} > \type{S}
                    }
                    \bussAx[\ruleLabel{typ-var}]{
                        \term{x}:\type{S} \vdashM \term{x}: \type{S}
                    }
                    \bussAx[\ruleLabel{typ-var}]{
                        \term{y}:\type{\ol{S}} \vdashM \term{y}: \type{\ol{S}}
                    }
                    \bussBin[\ruleLabel{typ-pair}]{
                        \term{x}:\type{S}, \term{y}:\type{\ol{S}} \vdashM \term{(x,y)}: \type{S \times \ol{S}}
                    }
                    \bussUn[\ruleLabel{typ-main/child}]{
                        \term{x}:\type{S}, \term{y}:\type{\ol{S}} \vdashC{\phi} \term{\phi\,(x,y)}: \type{S \times \ol{S}}
                    }
                    \bussBin[\ruleLabel{typ-res}]{
                        \type{\emptyset} \vdashC{\phi} \term{\pRes{x\tBfr{\epsi}y}(\phi\,(x,y))}: \type{S \times \ol{S}}
                    }
                \end{bussproof}
            \end{align*}

        \item
            Rule~\ruleLabel{red-spawn}:
            $\tCtx{\hat{F}}[\tSpawn M ; N] \reddC \term{\tCtx{\hat{F}}[N] \prl \tChild M}$.
            By induction on the structure of~$\tCtx{\hat{F}}$, which excludes holes under explicit substitution and so no names of $\term{M}$ are captured by~$\tCtx{\hat{F}}$.
            The inductive cases follows from the IH directly; we detail the base case \mbox{($\tCtx{\hat{F}} = \term{\phi \, \tHole}$)}:
            \begin{align*}
                & \begin{bussproof}
                    \bussAssume{
                        \type{\Gamma} \vdashM \term{M} : \type{\1}
                    }
                    \bussAssume{
                        \type{\Delta} \vdashM \term{N} : \type{T}
                    }
                    \bussBin[\ruleLabel{typ-spawn}]{
                        \type{\Gamma} , \type{\Delta} \vdashM \term{\tSpawn M ; N} : \type{T}
                    }
                    \bussUn[\ruleLabel{typ-main/-child}]{
                        \type{\Gamma} , \type{\Delta} \vdashC{\phi} \term{\phi \, \tSpawn M ; N} : \type{T}
                    }
                \end{bussproof}
                \\
                & \reddC
                \\
                & \begin{bussproof}
                    \bussAssume{
                        \type{\Delta} \vdashM \term{N} : \type{T}
                    }
                    \bussUn[\ruleLabel{typ-main/-child}]{
                        \type{\Delta} \vdashC{\phi} \term{\phi \, N} : \type{T}
                    }
                    \bussAssume{
                        \type{\Gamma} \vdashM \term{M} : \type{\1}
                    }
                    \bussUn[\ruleLabel{typ-child}]{
                        \type{\Gamma} \vdashC{\tChild} \term{\tChild \, M} : \type{\1}
                    }
                    \bussBin[\ruleLabel{typ-par-R}]{
                        \type{\Gamma} , \type{\Delta} \vdashC{\phi+\tChild} \term{\phi \, N \prl \tChild \, M} : \type{T}
                    }
                \end{bussproof}
            \end{align*}
            Clearly, $\term{\phi} + \term{\tChild} = \term{\phi}$, proving the thesis.

        \item
            \proofSrSend{app-}

        \item
            Rule~\ruleLabel{red-recv}: $\term{\pRes{x\tBfr{\vec{m},M}y} ( \tCtx{\hat{F}}[\tRecv y] \prl C )} \reddC \term{\pRes{x\tBfr{\vec{m}}y} ( \tCtx{\hat{F}}[(M,y)] \prl C )}$.

            For this case, we apply induction on the structure of $\tCtx{\hat{F}}$.
            The inductive cases follow from the IH directly.
            We consider the base case ($\tCtx{\hat{F}} = \term{\phi\,\tHole}$).
            We apply inversion of typing, w.l.o.g.\ assuming that $\term{\phi} = \term{\tMain}$, and then derive the typing of the reduced configuration:
            \begin{align*}
                & \pi \deq \begin{bussproof}
                    \bussAx[\ruleLabel{typ-var}]{
                        \term{y}:\type{{?}T . \ol{S}} \vdashM \term{y}: \type{{?}T . \ol{S}}
                    }
                    \bussUn[\ruleLabel{typ-recv}]{
                        \term{y}:\type{{?}T . \ol{S}} \vdashM \term{\tRecv y}: \type{T \times \ol{S}}
                    }
                    \bussUn[\ruleLabel{typ-main}]{
                        \term{y}:\type{{?}T . \ol{S}} \vdashC{\tMain} \term{\tMain\,(\tRecv y)}: \type{T \times \ol{S}}
                    }
                    \bussAssume{
                        \type{\Lambda}, \term{x}:\type{S'} \vdashC{\tChild} \term{C}: \type{\1}
                    }
                    \bussBin[\ruleLabel{typ-par-R}]{
                        \type{\Lambda}, \term{x}:\type{S'}, \term{y}:\type{{?}T . \ol{S}} \vdashC{\tMain} \term{\tMain\,(\tRecv y) \prl C}: \type{T \times \ol{S}}
                    }
                \end{bussproof}
                \\
                & \begin{bussproof}
                    \bussAssume{
                        \type{\Gamma} \vdashM \term{M}: \type{T}
                    }
                    \bussAssume{
                        \type{\Delta} \vdashB \term{\tBfr{\vec{m}}}: \type{S'} > \type{S}
                    }
                    \bussBin[\ruleLabel{typ-buf-send}]{
                        \type{\Gamma}, \type{\Delta} \vdashB \term{\tBfr{\vec{m},M}}: \type{S'} > \type{{!}T . S}
                    }
                    \bussAssume{
                        \pi
                    }
                    \bussBin[\ruleLabel{typ-res}]{
                        \type{\Gamma}, \type{\Delta}, \type{\Lambda} \vdashC{\tMain} \term{\pRes{x\tBfr{\vec{m},M}y} ( \tMain\,(\tRecv y) \prl C )}: \type{T \times \ol{S}}
                    }
                \end{bussproof}
                \\
                & \reddC
                \\
                & \begin{bussproof}
                    \def\defaultHypSeparation{\hskip1ex}
                    \def\ScoreOverhang{1pt}
                    \bussAssume{
                        \type{\Delta} \vdashB \term{\tBfr{\vec{m}}}: \type{S'} > \type{S}
                    }
                    \bussAssume{
                        \type{\Gamma} \vdashM \term{M}: \type{T}
                    }
                    \bussAx[\ruleLabel{typ-var}]{
                        \term{y}:\type{\ol{S}} \vdashM \term{y}: \type{\ol{S}}
                    }
                    \bussBin[\ruleLabel{typ-pair}]{
                        \type{\Gamma}, \term{y}:\type{\ol{S}} \vdashM \term{(M,y)}: \type{T \times \ol{S}}
                    }
                    \bussUn[\ruleLabel{typ-main}]{
                        \type{\Gamma}, \term{y}:\type{\ol{S}} \vdashC{\tMain} \term{\tMain\,(M,y)}: \type{T \times \ol{S}}
                    }
                    \bussAssume{
                        \type{\Lambda}, \term{x}:\type{S'} \vdashC{\tChild} \term{C}: \type{\1}
                    }
                    \bussBin[\ruleLabel{typ-par-R}]{
                        \type{\Gamma}, \type{\Lambda}, \term{x}:\type{S'}, \term{y}:\type{\ol{S}} \vdashC{\tMain} \term{\tMain\,(M,y) \prl C}: \type{T \times \ol{S}}
                    }
                    \bussBin[\ruleLabel{typ-res}]{
                        \type{\Gamma}, \type{\Delta}, \type{\Lambda} \vdashC{\tMain} \term{\pRes{x\tBfr{\vec{m}}y} ( \tMain\,(M,y) \prl C )}: \type{T \times \ol{S}}
                    }
                \end{bussproof}
            \end{align*}

        \item
            Rule~\ruleLabel{red-select} is similar to the case of Rule~\ruleLabel{red-send}:
            $\term{\pRes{x\tBfr{\vec{m}}y} ( \tCtx{F}[\tSel \ell \, x] \prl C )} \reddC \term{\pRes{x\tBfr{\ell,\vec{m}}y} ( \tCtx{F}[x] \prl C )}$.

        \item
            Rule~\ruleLabel{red-case} is similar to the case of Rule~\ruleLabel{red-recv}:
            \[
                j \in I \implies \term{\pRes{x\tBfr{\vec{m},j}y} ( \tCtx{F}[\tCase y \tOf \{ i : M_i \}_{i \in I}] \prl C )} \reddC \term{\pRes{x\tBfr{\vec{m}}y} ( \tCtx{F}[M_j~y] \prl C )}
            \]

        \item
            Rule~\ruleLabel{red-close}:
            $\term{\pRes{x\tBfr{\vec{m}}y} ( \tCtx{F}[\tClose x ; M] \prl C )} \reddC \term{\pRes{\tNil\tBfr{\vec{m}}y} ( \tCtx{F}[M] \prl C )}$.
            By induction on the structure of $\tCtx{F}$.
            The inductive cases follows from the IH straightforwardly; we detail the base case ($\tCtx{F} = \term{\phi \, \tHole}$).
            Assume, w.l.o.g., that $\term{\phi} = \term{\tMain}$.
            We first derive the typing of the configuration before reduction:
            \[
                \begin{bussproof}
                    \def\defaultHypSeparation{\hskip2mm}
                    \bussAssume{
                        \type{\Gamma} \vdashB \term{\tBfr{\vec{m}}} : \type{\tEnd}>\type{S}
                    }
                    \bussAssume{
                        \term{x}:\type{\tEnd} \vdashM \term{x} : \type{\tEnd}
                    }
                    \bussAssume{
                        \type{\Delta} \vdashM \term{M} : \type{T}
                    }
                    \bussBin[\ruleLabel{typ-close}]{
                        \type{\Delta} , \term{x}:\type{\tEnd} \vdashM \term{\tClose x ; M} : \type{T}
                    }
                    \bussUn[\ruleLabel{typ-main}]{
                        \type{\Delta} , \term{x}:\type{\tEnd} \vdashC{\tMain} \term{\tMain \, \tClose x ; M} : \type{T}
                    }
                    \bussAssume{
                        \type{\Theta} \vdashC{\tChild} \term{C} : \type{\1}
                    }
                    \bussBin[\ruleLabel{typ-par-R}]{
                        \type{\Delta} , \type{\Theta} , \term{x}:\type{\tEnd} \vdashC{\tMain} \term{\tMain \, \tClose x ; M \prl C} : \type{T}
                    }
                    \bussBin[\ruleLabel{typ-res/-buf}]{
                        \type{\Gamma'} \vdashC{\tMain} \term{\pRes{x\tBfr{\vec{m}}y} ( \tMain \, \tClose x ; M \prl C )} : \type{T}
                    }
                \end{bussproof}
            \]
            We prove that $\type{\Gamma} \vdashB \term{\tBfr{\vec{m}}} : \type{\tNil}>\type{S}$.
            The analysis depends on whether $\term{y} = \term{\tNil}$.
            \begin{itemize}

                \item
                    If $\term{y} = \term{\tNil}$, then $\term{\vec{m}} = \term{\epsi}$, $\type{S} = \type{\tNil}$, $\type{\Gamma'} = \type{\Gamma} , \type{\Delta} , \type{\Theta}$, and $\type{\Gamma} = \type{\emptyset}$, because $\type{\Gamma} \vdashB \term{\tBfr{\vec{m}}} : \type{\tEnd}>\type{S}$ must be derived as follows:
                    \[
                        \begin{bussproof}
                            \bussAx[\ruleLabel{typ-buf-end-L}]{
                                \type{\emptyset} \vdashB \term{\tBfr{\epsi}} : \type{\tEnd}>\type{\tNil}
                            }
                        \end{bussproof}
                    \]
                    The thesis holds as follows:
                    \[
                        \begin{bussproof}
                            \bussAx[\ruleLabel{typ-buf}]{
                                \type{\emptyset} \vdashB \term{\tBfr{\epsi}} : \type{\tNil}>\type{\tNil}
                            }
                        \end{bussproof}
                    \]

                \item
                    If $\term{y} \neq \term{\tNil}$, then $\type{\Gamma'} , \term{y}:\type{\ol{S}} = \type{\Gamma} , \type{\Delta} , \type{\Theta}$.
                    We apply induction on the size of $\term{\vec{m}}$ (\ih2):
                    \begin{itemize}

                        \item
                            If $\term{\vec{m}} = \term{\epsi}$, then $\type{\Gamma} = \type{\emptyset}$, $\type{S} = \type{\tEnd}$, because $\type{\Gamma} \vdashB \term{\tBfr{\epsi}} : \type{\tEnd}>\type{S}$ must be derived as follows:
                            \[
                                \begin{bussproof}
                                    \bussAx[\ruleLabel{typ-buf}]{
                                        \type{\emptyset} \vdashB \term{\tBfr{\epsi}} : \type{\tEnd}>\type{\tEnd}
                                    }
                                \end{bussproof}
                            \]
                            The thesis then holds as follows:
                            \[
                                \begin{bussproof}
                                    \bussAx[\ruleLabel{typ-buf-end-R}]{
                                        \type{\emptyset} \vdashB \term{\tBfr{\epsi}} : \type{\tNil}>\type{\tEnd}
                                    }
                                \end{bussproof}
                            \]

                        \item
                            If $\term{\vec{m}} = \term{\vec{m}',N}$, then $\type{\Gamma} = \type{\Gamma_1} , \type{\Gamma_2}$, and $\type{S} = \type{{!}U.S'}$, because $\type{\Gamma} \vdashB \term{\tBfr{\vec{m}',N}} : \type{\tEnd}>\type{S}$ must be derived as follows:
                            \[
                                \begin{bussproof}
                                    \bussAssume{
                                        \type{\Gamma_1} \vdashM \term{N} : \type{U}
                                    }
                                    \bussAssume{
                                        \type{\Gamma_2} \vdashB \term{\tBfr{\vec{m}'}} : \type{\tEnd}>\type{S'}
                                    }
                                    \bussBin[\ruleLabel{typ-buf-send}]{
                                        \type{\Gamma_1} , \type{\Gamma_2} \vdashB \term{\tBfr{\vec{m}',N}} : \type{\tEnd}>\type{{!}U.S'}
                                    }
                                \end{bussproof}
                            \]
                            By \ih2, $\type{\Gamma_2} \vdashB \term{\tBfr{\vec{m}'}} : \type{\tNil}>\type{S'}$.
                            The thesis then holds as follows:
                            \[
                                \begin{bussproof}
                                    \bussAssume{
                                        \type{\Gamma_1} \vdashM \term{N} : \type{U}
                                    }
                                    \bussAssume{
                                        \type{\Gamma_2} \vdashB \term{\tBfr{\vec{m}'}} : \type{\tNil}>\type{S'}
                                    }
                                    \bussBin[\ruleLabel{typ-buf-send}]{
                                        \type{\Gamma_1} , \type{\Gamma_2} \vdashB \term{\tBfr{\vec{m}',N}} : \type{\tNil}>\type{{!}U.S'}
                                    }
                                \end{bussproof}
                            \]
                    \end{itemize}

            \end{itemize}
            Now, we derive the typing of the reduced configuration:
            \[
                \begin{bussproof}
                    \bussAssume{
                        \type{\Gamma} \vdashB \term{\tBfr{\vec{m}}} : \type{\tNil}>\type{S}
                    }
                    \bussAssume{
                        \type{\Delta} \vdashM \term{M} : \type{T}
                    }
                    \bussUn[\ruleLabel{typ-main}]{
                        \type{\Delta} \vdashC{\tMain} \term{\tMain \, M} : \type{T}
                    }
                    \bussAssume{
                        \type{\Theta} \vdashC{\tChild} \term{C} : \type{\1}
                    }
                    \bussBin[\ruleLabel{typ-par-R}]{
                        \type{\Delta} , \type{\Theta} \vdashC{\tMain} \term{\tMain \, M \prl C} : \type{T}
                    }
                    \bussBin[\ruleLabel{typ-res/-buf}]{
                        \type{\Gamma'} \vdashC{\tMain} \term{\pRes{\tNil\tBfr{\vec{m}}y} ( \tMain \, M \prl C )} : \type{T}
                    }
                \end{bussproof}
            \]

        \item
            Rule~\ruleLabel{red-res-nil}:
            $\term{\pRes{\tNil\tBfr{\epsi}\tNil} C} \reddC \term{C}$.
            We have
            \begin{align*}
                & \begin{bussproof}
                    \bussAx[\ruleLabel{typ-buf}]{
                        \type{\emptyset} \vdashB \term{\tBfr{\epsi}} : \type{\tNil}>\type{\tNil}
                    }
                    \bussAssume{
                        \type{\Gamma} \vdashC{\phi} \term{C} : \type{T}
                    }
                    \bussBin[\ruleLabel{typ-res/-res-buf}]{
                        \type{\Gamma} \vdashC{\phi} \term{\pRes{\tNil\tBfr{\epsi}\tNil} C} : \type{T}
                    }
                \end{bussproof}
                \\
                & \reddC
                \\
                & \type{\Gamma} \vdashC{\phi} \term{C} : \type{T}
            \end{align*}

        \item
            Rule~\ruleLabel{red-par-nil}:
            $\term{C \prl \tChild \, ()} \reddC \term{C}$.
            We have the following:
            \begin{align*}
                & \begin{bussproof}
                    \bussAssume{
                        \type{\Gamma} \vdashC{\phi} \term{C} : \type{T}
                    }
                    \bussAx[\ruleLabel{typ-unit}]{
                        \type{\emptyset} \vdashM \term{()} : \type{\1}
                    }
                    \bussUn[\ruleLabel{typ-child}]{
                        \type{\emptyset} \vdashC{\tChild} \term{\tChild \, ()} : \type{\1}
                    }
                    \bussBin[\ruleLabel{typ-par}]{
                        \type{\Gamma} \vdashC{\phi} \term{C \prl \tChild \, ()} : \type{T}
                    }
                \end{bussproof}
                \\
                & \reddC
                \\
                & \type{\Gamma} \vdashC{\phi} \term{C} : \type{T}
                \tag*{\qedhere}
            \end{align*}

    \end{itemize}
\end{proof}


\subsection{Translation: Type Preservation}
\label{as:LASTn:transTp}


\tLASTnTransTp*

\begin{proof}
    By induction on the \LASTn typing derivation.
    It is sufficient to give the typing of the translations of typing rules in \Cref{f:LASTn:LASTn:trans1,f:LASTn:LASTn:trans2,f:LASTn:LASTn:trans3} as follows; checking the derivations is straightforward.
    \begin{align*}
        & \tinyRL{typ-var}
        && \vdash* \encc{z}{\normalcolor \term{x}:\type{T} \vdashM \term{x} : \type{T}}
        && \typInf x:\enct*{T} , z:\enct{T}
        \\ \displaybreak[1]
        & \tinyRL{typ-abs}
        && \vdash* \encc{z}{\normalcolor \type{\Gamma} \vdashM \term{\lam x . M} : \type{T \lolli U}}
        && \typInf \enct*{\Gamma} , z:\enct*{T} \parr \enct{U} = \enct{T \lolli U}
        \\ \displaybreak[1]
        & \tinyRL{typ-app}
        && \vdash* \encc{z}{\normalcolor \type{\Gamma} , \type{\Delta} \vdashM \term{M~N} : \type{U}}
        && \typInf \enct*{\Gamma} , \enct*{\Delta} , z:\enct{U}
        \\ \displaybreak[1]
        & \tinyRL{typ-unit}
        && \vdash* \encc{z}{\normalcolor \type{\emptyset} \vdashM \term{()} : \type{\1}}
        && \typInf z:\bullet = \enct{\1}
        \\ \displaybreak[1]
        & \tinyRL{typ-pair}
        && \vdash* \encc{z}{\normalcolor \type{\Gamma} , \type{\Delta} \vdashM \term{(M,N)} : \type{T \times U}}
        && \typInf \enct*{\Gamma} , \enct*{\Delta} , z:\begin{array}[t]{@{}l@{}}
            (\bullet \parr \enct{T}) \tensor (\bullet \parr \enct{U})
            \\
            = \ol{\enct*{T}} \tensor \ol{\enct*{U}} \\
            = \enct{T \times U}
        \end{array}
        \\ \displaybreak[1]
        & \tinyRL{typ-split}
        && \vdash* \encc{z}{\normalcolor \type{\Gamma} , \type{\Delta} \vdashM \term{\tLet (x,y) = M \tIn N} : \type{U}}
        && \typInf \enct*{\Gamma} , \enct*{\Delta} , z:\enct{U}
        \\ \displaybreak[1]
        & \tinyRL{typ-new}
        && \vdash* \encc{z}{\normalcolor \type{\emptyset} \vdashM \term{\tNew} : \type{S \times \ol{S}}}
        && \typInf z:\ol{\enct*{S}} \tensor \ol{\enct*{\ol{S}}} = \enct{S \times \ol{S}}
        \\ \displaybreak[1]
        & \tinyRL{typ-spawn}
        && \vdash* \encc{z}{\normalcolor \type{\Gamma} , \type{\Delta} \vdashM \term{\tSpawn M ; N} : \type{T}}
        && \typInf \enct*{\Gamma} , \enct*{\Delta} , z:\enct{T}
        \\ \displaybreak[1]
        & \tinyRL{typ-close}
        && \vdash* \encc{z}{\normalcolor \type{\Gamma} , \type{\Delta} \vdashM \term{\tClose M ; N} : \type{T}}
        && \typInf \enct*{\Gamma} , \enct*{\Delta} , z:\enct{T}
        \\ \displaybreak[1]
        & \tinyRL{typ-send}
        && \vdash* \encc{z}{\normalcolor \type{\Gamma} \vdashM \term{\tSend M \, N} : \type{S}}
        && \typInf \enct*{\Gamma} , \enct*{\Delta} , z:\enct{S}
        \\ \displaybreak[1]
        & \tinyRL{typ-recv}
        && \vdash* \encc{z}{\normalcolor \type{\Gamma} \vdashM \term{\tRecv M} : \type{T \times S}}
        && \typInf \enct*{\Gamma} , z:\begin{array}[t]{@{}l@{}}
            (\bullet \parr \enct{T}) \tensor (\bullet \parr \enct{S})
            \\
            = \ol{\enct*{T}} \tensor \ol{\enct*{S}} = \enct{T \times S}
        \end{array}
        \\ \displaybreak[1]
        & \tinyRL{typ-sel}
        && \vdash* \encc{z}{\normalcolor \type{\Gamma} \vdashM \term{\tSel j \, M} : \type{S_j}}
        && \typInf \enct*{\Gamma} , z:\enct{S_j}
        \\ \displaybreak[1]
        & \tinyRL{typ-case}
        && \vdash* \encc{z}{\normalcolor \type{\Gamma} \vdashM \term{\tCase M \tOf \{ i : N_i \}_{i \in I}} : \type{U}}
        && \typInf \enct*{\Gamma} , \enct*{\Delta} , z:\enct{U}
        \\ \displaybreak[1]
        & \tinyRL{typ-sub}
        && \vdash* \encc{z}{\normalcolor \type{\Gamma} , \type{\Delta} \vdashM \term{M \tSub{ N/x }} : \type{U}}
        && \typInf \enct*{\Gamma} , \enct*{\Delta} , z:\enct{U}
        \\ \displaybreak[1]
        & \tinyRL{typ-main}
        && \vdash* \encc{z}{\normalcolor \type{\Gamma} \vdashC{\tMain} \term{\tMain \, M} : \type{\hat{T}}}
        && \typInf \enct*{\Gamma} , z:\enct{\hat{T}}
        \\ \displaybreak[1]
        & \tinyRL{typ-child}
        && \vdash* \encc{z}{\normalcolor \type{\Gamma} \vdashC{\tChild} \term{\tChild \, M} : \type{\1}}
        && \typInf \enct*{\Gamma} , z:\bullet = \enct{\1}
        \\ \displaybreak[1]
        & \tinyRL{typ-par}~(\type{T_1} = \type{\1})
        && \vdash* \encc{z}{\normalcolor \type{\Gamma} , \type{\Delta} \vdashC{\phi_1+\phi_2} \term{C \prl D} : \type{T_2}}
        && \typInf \enct*{\Gamma} , \enct*{\Delta} , z:\enct{T_2}
        \\ \displaybreak[1]
        & \tinyRL{typ-par}~(\type{T_2} = \type{\1})
        && \vdash* \encc{z}{\normalcolor \type{\Gamma} , \type{\Delta} \vdashC{\phi_1+\phi_2} \term{C \prl D} : \type{T_1}}
        && \typInf \enct*{\Gamma} , \enct*{\Delta} , z:\enct{T_1}
        \\ \displaybreak[1]
        & \tinyRL{typ-res}
        && \vdash* \encc{z}{\normalcolor \type{\Gamma'} \vdashC{\phi} \term{\pRes{x\tBfr{\vec{m}}y} C} : \type{T}}
        && \typInf \enct*{\Gamma'} , z:\enct{T}
        \\ \displaybreak[1]
        & \tinyRL{typ-conf-sub}
        && \vdash* \encc{z}{\normalcolor \type{\Gamma} , \type{\Delta} \vdashC{\phi} \term{C \tSub{ M/x }} : \type{U}}
        && \typInf \enct*{\Gamma} , \enct*{\Delta} , z:\enct{U}
        \\ \displaybreak[1]
        & \begin{array}[t]{@{}l@{}}
            \tinyRL{typ-buf}
            \\
            (\type{S'} = \type{{!}T.S})
        \end{array}
        && \vdash* \encc{a\>b}{\normalcolor \type{\emptyset} \vdashB \term{\tBfr{\epsi}} : \type{S'}>\type{S'}}
        && \typInf \begin{array}[t]{@{}l@{}}
            a:\begin{array}[t]{@{}l@{}}
                \bullet \parr \bullet \tensor \enct*{T} \parr \ol{\enct*{S}} \\
                = \bullet \parr \enct{S'} = \ol{\enct*{S'}} ,
            \end{array}
            \\
            b:\begin{array}[t]{@{}l@{}}
                \bullet \parr \ol{\enct*{T}} \tensor \ol{\enct*{\ol{S}}}
                \\
                = \bullet \parr \enct{{?}T.\ol{S}} = \bullet \parr \enct{\ol{S'}} = \ol{\enct*{\ol{S'}}}
            \end{array}
        \end{array}
        \\ \displaybreak[1]
        & \begin{array}[t]{@{}l@{}}
            \tinyRL{typ-buf}
            \\
            \mathrlap{
                (\type{S'} = \type{\oplus \{ i : S_i \}_{i \in I}})
            }
        \end{array}
        && \vdash* \encc{a\>b}{\normalcolor \type{\emptyset} \vdashB \term{\tBfr{\epsi}} : \type{S'}>\type{S'}}
        && \typInf \begin{array}[t]{@{}l@{}}
            a:\begin{array}[t]{@{}l@{}}
                \bullet \parr \bullet \tensor \& \{ i : \ol{\enct*{S_i}} \}_{i \in I} \\
                = \bullet \parr \enct{S'} = \ol{\enct*{S'}} ,
            \end{array}
            \\
            b:\begin{array}[t]{@{}l@{}}
                \bullet \parr \oplus \{ i : \ol{\enct*{\ol{S_i}}} \}_{i \in I} \\
                = \bullet \parr \enct{\& \{ i : \ol{S_i} \}_{i \in I}} \\
                = \bullet \parr \enct{\ol{S'}} = \ol{\enct*{\ol{S'}}}
            \end{array}
        \end{array}
        \\ \displaybreak[1]
        & \begin{array}[t]{@{}l@{}}
            \tinyRL{typ-buf}
            \\
            \mathrlap{
                (\type{S'} \in \{\type{{?}T.S},\type{\& \{ i : S_i \}_{i \in I}}\})
            }
        \end{array}
        && \vdash* \encc{a\>b}{\normalcolor \type{\emptyset} \vdashB \term{\tBfr{\epsi}} : \type{S'}>\type{S'}}
        && \typInf a:\ol{\enct*{\ol{\ol{S'}}}} = \ol{\enct*{S'}} , b:\ol{\enct*{\ol{S'}}}
        \\ \displaybreak[1]
        & \begin{array}[t]{@{}l@{}}
            \tinyRL{typ-buf}
            \\
            (\type{S'} = \type{\tEnd})
        \end{array}
        && \vdash* \encc{a\>b}{\normalcolor \type{\emptyset} \vdashB \term{\tBfr{\epsi}} : \type{S'}>\type{S'}}
        && \typInf \begin{array}[t]{@{}l@{}}
            a:\bullet \parr \bullet \tensor \bullet = \ol{\enct*{S'}} , \\
            b:\bullet \parr \bullet \tensor \bullet = \ol{\enct*{S'}} = \ol{\enct*{\ol{S'}}}
        \end{array}
        \\ \displaybreak[1]
        & \begin{array}[t]{@{}l@{}}
            \tinyRL{typ-buf}
            \\
            (\type{S'} = \type{\tNil})
        \end{array}
        && \vdash* \encc{a\>b}{\normalcolor \type{\emptyset} \vdashB \term{\tBfr{\epsi}} : \type{S'}>\type{S'}}
        && \typInf a:\bullet = \ol{\enct*{S'}} , b:\bullet = \ol{\enct*{S'}} = \ol{\enct*{\ol{S'}}}
        \\ \displaybreak[1]
        & \tinyRL{typ-buf-send}
        && \vdash* \encc{a\>b}{\normalcolor \type{\Gamma} , \type{\Delta} \vdashB \term{\tBfr{\vec{m},M}} : \type{S'}>\type{{!}T.S}}
        && \typInf \begin{array}[t]{@{}l@{}}
            \enct*{\Gamma} , \enct*{\Delta} , a:\ol{\enct*{S'}} , \\
            b:\begin{array}[t]{@{}l@{}}
                \bullet \parr \ol{\enct*{T}} \tensor \ol{\enct*{\ol{S}}} = \bullet \parr \enct{{?}T.\ol{S}}
                \\
                = \bullet \parr \enct{\ol{{!}T.S}} = \ol{\enct*{\ol{{!}T.S}}}
            \end{array}
        \end{array}
        \\ \displaybreak[1]
        & \tinyRL{typ-buf-sel}
        && \vdash* \encc{a\>b}{\normalcolor \type{\Gamma} \vdashB \term{\tBfr{\vec{m},j}} : \type{S'}>\type{\oplus \{ i : S_i \}_{i \in I}}}
        && \typInf \enct*{\Gamma} , a:\ol{\enct*{S'}} , b:\begin{array}[t]{@{}l@{}}
            \bullet \parr \oplus \{ i : \ol{\enct*{\ol{S_i}}} \}_{i \in I}
            \\
            = \bullet \parr \enct{\& \{ i : \ol{S_i} \}_{i \in I}}
            \\
            = \bullet \parr \enct{\ol{\oplus \{ i : S_i \}_{i \in I}}}
            \\
            = \ol{\enct*{\ol{\oplus \{ i : S_i \}_{i \in I}}}}
        \end{array}
        \\ \displaybreak[1]
        & \tinyRL{typ-buf-end-L}
        && \vdash* \encc{a\>b}{\normalcolor \type{\emptyset} \vdashB \term{\tBfr{\epsi}} : \type{\tEnd}>\type{\tNil}}
        && \typInf \begin{array}[t]{@{}l@{}}
            a:\bullet \parr \bullet \tensor \bullet = \ol{\enct*{\tEnd}} , \\
            b:\bullet = \ol{\enct*{\tNil}} = \ol{\enct*{\ol{\tNil}}}
        \end{array}
        \\ \displaybreak[1]
        & \tinyRL{typ-buf-end-R}
        && \vdash* \encc{a\>b}{\normalcolor \type{\emptyset} \vdashB \term{\tBfr{\epsi}} : \type{\tNil}>\type{\tEnd}}
        && \typInf \begin{array}[t]{@{}l@{}}
            a:\bullet = \ol{\enct*{\tNil}} , \\
            b:\bullet \parr \bullet \tensor \bullet = \ol{\enct*{\tEnd}} = \ol{\enct*{\ol{\tEnd}}}
        \end{array}
        \tag*{\qedhere}
    \end{align*}
\end{proof}


\subsection{Operational Correspondence}
\label{as:LASTn:oc}


\Cref{as:LASTn:oc:compl,as:LASTn:oc:sound} prove completeness and soundness, respectively.
Both results rely on the following lemma, which ensures that \LASTn contexts translate to evaluation contexts in \APCP.

\lLASTnTransCtxs*

\begin{sketch}
    By induction on the structure of the contexts.
\end{sketch}


\subsubsection{Completeness}
\label{as:LASTn:oc:compl}


Here we prove the completeness of the translation.
The proof relies on the following intermediate results:
\begin{itemize}

    \item
        \Cref{t:LASTn:LASTn:transTermSc,t:LASTn:LASTn:transConfSc} prove that the translation preserves structural congruence for terms and configurations, respectively.

    \item
        \Cref{t:LASTn:LASTn:transTermRed} then shows that the translation is complete with respect to term reduction.

    \item
        \Cref{l:LASTn:LASTn:transBuf} shows how the translation of buffers with messages can be decomposed using evaluation contexts.

    \item
        Finally, we prove completeness (\Cref{t:LASTn:LASTn:transCompl}).

\end{itemize}

\begin{restatable}[Preservation of Structural Congruence for Terms]{theorem}{tLASTnTransTermSc}
\label{t:LASTn:LASTn:transTermSc}
    Given $\type{\Gamma} \vdashM \term{M}: \type{T}$, if $\term{M} \equivM \term{N}$, then $\encc{z}{M} \equiv \encc{z}{N}$.
\end{restatable}

\begin{proof}
    By induction on the derivation of $\term{M} \equivM \term{N}$ (\ih1).
    \proofLASTnTransTermSc12{\qedhere}
\end{proof}

\begin{restatable}[Preservation of Structural Congruence for Configurations]{theorem}{tLASTnTransConfSc}
    \label{t:LASTn:LASTn:transConfSc}
    Given $\type{\Gamma} \vdashC{\phi} \term{C}: \type{T}$, if $\term{C} \equivC \term{D}$, then $\encc{z}{C} \equiv \encc{z}{D}$.
\end{restatable}

\begin{proof}
    By induction on the derivation of $\term{C} \equivC \term{D}$ (\ih1).
    The inductive cases follow from \ih1 and \Cref{l:LASTn:LASTn:transCtxs} straightforwardly.
    We detail the base cases, induced by the ten rules in \Cref{f:LASTn:LASTn:confs}:
    \begin{itemize}

        \item
            Rule~\ruleLabel{sc-term-sc}:
            $\term{M} \equivM \term{M'}$ implies $\term{\phi \, M} \equivC \term{\phi \, M'}$.
            We have $\encc{z}{\phi \, M} = \encc{z}{M}$ and $\encc{z}{\phi \, M'} = \encc{z}{M'}$.
            By the assumption that $\term{M} \equivM \term{M'}$ and \Cref{t:LASTn:LASTn:transTermSc}, $\encc{z}{M} \equiv \encc{z}{M'}$.
            The thesis follows immediately.

        \item
            \proofLASTnTransConfScResSwap{}{2}

        \item
            Rule~\ruleLabel{sc-res-comm}:
            $\term{\pRes{x\tBfr{\vec{m}}y} \pRes{z\tBfr{\vec{n}}w} C} \equivC \term{\pRes{z\tBfr{\vec{n}}w} \pRes{x\tBfr{\vec{m}}y} C}$.
            The thesis holds as follows:
            \begin{align*}
                & \encc{z}{\pRes{x\tBfr{\vec{m}}y} \pRes{z\tBfr{\vec{n}}w} C} 
                \\
                &= \pRes{a_1x} \pRes{b_1y} \big( \encc{a_1\>b_1}{\tBfr{\vec{m}}} \| \pRes{a_2z} \pRes{b_2w} ( \encc{a_2\>b_2}{\tBfr{\vec{n}}} \| \encc{z}{C} ) \big)
                \\
                &\equiv \pRes{a_2z} \pRes{b_2w} \big( \encc{a_2\>b_2}{\tBfr{\vec{n}}} \| \pRes{a_1x} \pRes{b_1y} ( \encc{a_1\>b_1}{\tBfr{\vec{m}}} \| \encc{z}{C} ) \big)
                \\
                &= \encc{z}{\pRes{z\tBfr{\vec{n}}w} \pRes{x\tBfr{\vec{m}}y} C}
            \end{align*}

        \item
            Rule~\ruleLabel{sc-res-ext}:
            $\term{x},\term{y} \notin \fv(\term{C})$ implies $\term{\pRes{x\tBfr{\vec{m}}y} ( C \prl D )} \equivC \term{C \prl \pRes{x\tBfr{\vec{m}}y} D}$.
            The analysis depends on which of $\term{C},\term{D}$ is a child thread; w.l.o.g., assume that $\term{C}$ is.
            Assume the condition; by \Cref{l:LASTn:LASTn:transCtxs}, then $x,y \notin \fn(\encc{\_}{C})$ ($\ast$).
            The thesis holds as follows:
            \begin{align*}
                \encc{z}{\term{\pRes{x\tBfr{\vec{m}}y} ( C \prl D )}} &= \pRes{ax} \pRes{by} ( \encc{a\>b}{\tBfr{\vec{m}}} \| \pRes{\_\_} \encc{\_}{C} \| \encc{z}{D} )
                \\
                &\equiv \pRes{\_\_} \encc{\_}{C} \| \pRes{ax} \pRes{by} ( \encc{a\>b}{\tBfr{\vec{m}}} \| \encc{z}{D} )
                \tag{$\ast$}
                \\
                &= \encc{z}{C \prl \pRes{x\tBfr{\vec{m}}y} D}
            \end{align*}

        \item
            Rule~\ruleLabel{sc-par-comm}:
            $\term{C \prl D} \equivC \term{D \prl C}$.
            The analysis depends on which of $\term{C},\term{D}$ is a child thread; w.l.o.g., assume that $\term{C}$ is.
            The thesis holds as follows:
            \begin{align*}
                \encc{z}{C \prl D} &= \pRes{\_\_} \encc{\_}{C} \| \encc{z}{D}
                \\
                &\equiv \encc{z}{D} \| \pRes{\_\_} \encc{\_}{C}
                \\
                &= \encc{z}{D \prl C}
            \end{align*}

        \item
            Rule~\ruleLabel{sc-par-assoc}:
            $\term{C \prl (D \prl E)} \equivC \term{(C \prl D) \prl E}$.
            The analysis depends on which of $\term{C},\term{D},\term{E}$ are child threads; w.l.o.g., assume that $\term{C},\term{D}$ are.
            The thesis holds as follows:
            \begin{align*}
                \encc{z}{C \prl (D \prl E)} &= \pRes{\_\_} \encc{\_}{C} \| ( \pRes{\_\_} \encc{\_\_}{D} \| \encc{z}{E} )
                \\
                &\equiv \pRes{\_\_} ( \pRes{\_\_} \encc{\_}{C} \| \encc{\_\_}{D} ) \| \encc{z}{E}
                \\
                &= \encc{z}{(C \prl D) \prl E}
            \end{align*}

        \item
            Rule~\ruleLabel{sc-conf-sub}:
            $\term{\phi \, (M \tSub{ N/x })} \equivC \term{(\phi \, M) \tSub{ N/x }}$.
            The thesis holds as follows:
            \begin{align*}
                \encc{z}{\phi \, (M \tSub{ N/x })} &= \pRes{xa} ( \encc{z}{M} \| \pIn a(\_,b) ; \encc{b}{N} )
                \\
                &= \pRes{xa} ( \encc{z}{\phi \, M} \| \pIn a(\_,b) ; \encc{b}{N} )
                \\
                &= \encc{z}{(\phi \, M) \tSub{ N/x }}
            \end{align*}

        \item
            Rule~\ruleLabel{sc-conf-sub-ext}:
            $\term{x} \notin \fv(\tCtx{G})$ implies $\term{(\tCtx{G}[C]) \tSub{ M/x }} \equivC \tCtx{G}[C \tSub{ M/x }]$.
            By \Cref{l:LASTn:LASTn:transCtxs}, for any $\term{D}$, $\encc{z}{\tCtx{G}[D]} = \evalCtx{E}[\encc{z'}{D}]$ for some $\evalCtx{E},z'$ ($\ast_1$).
            Assume the condition; by \Cref{l:LASTn:LASTn:transCtxs}, then $x \notin \fn(\evalCtx{E})$ ($\ast_2$).
            The thesis holds as follows:
            \begin{align*}
                \encc{z}{(\tCtx{G}[C]) \tSub{ M/x }} &= \pRes{xa} ( \encc{z}{\tCtx{G}[C]} \| \pIn a(\_,b) ; \encc{b}{M} )
                \\
                &= \pRes{xa} ( \evalCtx{E}[\encc{z'}{C}] \| \pIn a(\_,b) ; \encc{b}{M} )
                \tag{$\ast_1$}
                \\
                &\equiv \evalCtx{E}[\pRes{xa} ( \encc{z'}{C} \| \pIn a(\_,b) ; \encc{b}{M} )]
                \tag{$\ast_2$}
                \\
                &= \evalCtx{E}[\encc{z'}{C \tSub{ M/x }}]
                \\
                &= \encc{z}{\tCtx{G}[C \tSub{ M/x }]}
                \tag{$\ast_1$}
                \\
                \tag*{\qedhere}
            \end{align*}

    \end{itemize}
\end{proof}

\begin{restatable}[Completeness of Reduction for Terms]{theorem}{tLASTnTransTermRed}
    \label{t:LASTn:LASTn:transTermRed}
    Given $\type{\Gamma} \vdashM \term{M}: \type{T}$, if $\term{M} \reddM \term{N}$, then $\encc{z}{M} \redd* \encc{z}{N}$.
\end{restatable}

\begin{proof}
    By induction on the derivation of $\term{M} \reddM \term{N}$.
    We detail each rule:
    \begin{itemize}

        \item
            \proofTransTermRedLam

        \item
            Rule~\ruleLabel{red-pair}:
            $\term{\tLet (x,y) = (M_1,M_2) \tIn N} \reddM \term{N \tSub{ M_1/x,M_2/y }}$.
            The thesis holds as follows:
            \begin{align*}
                & \encc{z}{\tLet (x,y) = (M_1,M_2) \tIn N} 
                \\
                &= \pRes{a_1b_1} \big(
                    \begin{array}[t]{@{}l@{}}
                        \pIn a_1(x,y) ; \encc{z}{N}
                        \\
                        {} \| \pRes{a_2b_2} \pRes{c_2d_2} ( \begin{array}[t]{@{}l@{}}
                            \pOut b_1[a_2,c_2]
                            \\
                            {} \| \pIn b_2(\_,e_2) ; \encc{e_2}{M_1}
                            \\
                            {} \| \pIn d_2(\_,f_2) ; \encc{f_2}{M_2} )
                        \big)
                    \end{array}
                \end{array}
                \\
                &\redd \pRes{a_2b_2}\pRes{c_2d_2} ( \begin{array}[t]{@{}l@{}}
                    \encc{z}{N} \{ a_2/x,c_2/y \}
                    \\
                    {} \| \pIn b_2(\_,e_2) ; \encc{e_2}{M_1} \| \pIn d_2(\_,f_2) ; \encc{f_2}{M_2} )
                \end{array}
                \\
                &\equiv \pRes{yd_2} \big( \pRes{xb_2} ( \encc{z}{N} \| \pIn b_2(\_,e_2) ; \encc{e_2}{M_1} ) \| \pIn d_2(\_,f_2) ; \encc{f_2}{M_2} \big)
                \\
                &= \encc{z}{N \tSub{ M_1/x,M_2/y }}
            \end{align*}

        \item
            Rule~\ruleLabel{red-name-sub}:
            $\term{x \tSub{ M/x }} \reddM \term{M}$.
            The thesis holds as follows:
            \begin{align*}
                \encc{z}{x \tSub{ M/x }} &= \pRes{xa} ( \pOut x[\_,z] \| \pIn a(\_,b) ; \encc{b}{M} )
                \\
                &\redd \encc{z}{M}
            \end{align*}

        \item
            Rule~\ruleLabel{red-lift}:
            $\term{M} \reddM \term{N}$ implies $\tCtx{R}[M] \reddM \tCtx{R}[N]$.
            By \Cref{l:LASTn:LASTn:transCtxs}, for any~$\term{L}$, \mbox{$\encc{z}{\tCtx{R}[L]} = \evalCtx{E}[\encc{z'}{L}]$} for some $\evalCtx{E},z'$ ($\ast_1$).
            Assume the condition; by the IH, \mbox{$\encc{z'}{M} \redd* \encc{z'}{N}$} ($\ast_2$).
            The thesis holds as follows:
            \begin{align*}
                \encc{z}{\tCtx{R}[M]} &= \evalCtx{E}[\encc{z'}{M}]
                \tag{$\ast_1$}
                \\
                &\redd* \evalCtx{E}[\encc{z'}{N}]
                \tag{$\ast_2$}
                \\
                &= \encc{z}{\tCtx{R}[N]}
                \tag{$\ast_1$}
            \end{align*}

        \item
            Rule~\ruleLabel{red-lift-sc}:
            $\term{M} \equivM \term{M'}$, $\term{M'} \reddM \term{N'}$, and $\term{N'} \equivM \term{N}$ imply $\term{M} \reddM \term{N}$.
            Assume the conditions.
            By \Cref{t:LASTn:LASTn:transTermSc}, $\encc{z}{M} \equiv \encc{z}{M'}$ ($\ast_1$) and $\encc{z}{N'} \equiv \encc{z}{N}$ ($\ast_2$).
            By the IH, $\encc{z}{M'} \redd* \encc{z}{N'}$ ($\ast_3$).
            The thesis holds as follows:
            \begin{align*}
                \encc{z}{M} &\equiv \encc{z}{M'}
                \tag{$\ast_1$}
                \\
                &\redd* \encc{z}{N'}
                \tag{$\ast_3$}
                \\
                &\equiv \encc{z}{N}
                \tag{$\ast_2$}
                \\
                \tag*{\qedhere}
            \end{align*}

    \end{itemize}
\end{proof}

\lLASTnTransBuf*

\begin{sketch}
    By induction on the size of $\vec{m}$.
\end{sketch}

\tLASTnTransCompl*

\begin{proof}
    By induction on the derivation of $\term{C} \reddC \term{D}$.
    We detail every rule:
    \begin{itemize}

        \item
            Rule~\ruleLabel{red-new}:
            $\tCtx{F}[\tNew] \reddC \term{\pRes{x\tBfr{\epsi}y} ( \tCtx{F}[(x,y)] )}$.
            By \Cref{l:LASTn:LASTn:transCtxs}, for any $\term{M}$, $\encc{z}{\tCtx{F}[M]} = \evalCtx{E}[\encc{z'}{M}]$ for some $\evalCtx{E},z'$ ($\ast$).
            The thesis holds as follows:
            \begin{align*}
                \encc{z}{\tCtx{F}[\tNew]} &= \evalCtx{E}[\encc{z'}{\tNew}]
                \tag{$\ast$}
                \\
                &= \evalCtx{E}[\pRes{ab} \big( \pOut a[\_,z'] \| \pIn b(\_,c) ; \pRes{dx} \pRes{ey} ( \encc{d\>e}{\tBfr{\epsi}} \| \encc{c}{(x,y)} ) \big)]
                \\
                &\redd \evalCtx{E}[\pRes{dx} \pRes{ey} ( \encc{d\>e}{\tBfr{\epsi}} \| \encc{z'}{(x,y)} )]
                \\
                &\equiv \pRes{dx} \pRes{ey} ( \encc{d\>e}{\tBfr{\epsi}} \| \evalCtx{E}[\encc{z'}{(x,y)}] )
                \\
                &= \pRes{dx} \pRes{ey} ( \encc{d\>e}{\tBfr{\epsi}} \| \encc{z}{\tCtx{F}[(x,y)]} )
                \tag{$\ast$}
                \\
                &= \encc{z}{\pRes{x\tBfr{\epsi}y} (\tCtx{F}[(x,y)])}
            \end{align*}

        \item
            Rule~\ruleLabel{red-spawn}:
            $\tCtx{\hat{F}}[\tSpawn M ; N] \reddC \term{\tCtx{\hat{F}}[N] \prl \tChild \, M}$.
            By \Cref{l:LASTn:LASTn:transCtxs}, for any $\term{L}$, $\encc{z}{\tCtx{F}[L]} = \evalCtx{E}[\encc{z'}{L}]$ for some $\evalCtx{E},z'$ ($\ast_1$).
            Moreover, since $\tCtx{\hat{F}}$ does not have its hole under an explicit substitution, it does not capture any free variables of $M$; hence, by \Cref{l:LASTn:LASTn:transCtxs}, $\evalCtx{E}$ does not capture any free names of $\encc{u}{M}$ for any~$u$~($\ast_2$).
            The thesis holds as follows:
            \begin{align*}
                \encc{z}{\tCtx{\hat{F}}[\tSpawn M ; N]} &= \evalCtx{E}[\encc{z'}{\tSpawn M ; N}]
                \tag{$\ast_1$}
                \\
                &= \evalCtx{E}[\pRes{ab} \big( \pOut a[\_,z'] \| \pIn b(\_,c); ( \pRes{\_\_} \encc{\_}{M} \| \encc{c}{N} ) \big) ]
                \\
                &\redd \evalCtx{E}[\pRes{\_\_} \encc{\_}{M} \| \encc{z'}{N}]
                \\
                &\equiv \evalCtx{E}[\encc{z'}{N}] \| \pRes{\_\_} \encc{\_}{M}
                \tag{$\ast_2$}
                \\
                &= \encc{z}{\tCtx{\hat{F}}[N]} \| \pRes{\_\_} \encc{\_}{M}
                \tag{$\ast_1$}
                \\
                &= \encc{z}{\tCtx{\hat{F}}[N] \prl \tChild \, M}
            \end{align*}

        \item
            \proofTransComplSend

        \item
            Rule~\ruleLabel{red-recv}:
            $\term{\pRes{x\tBfr{\vec{m},M}y} ( \tCtx{\hat{F}}[\tRecv y] \prl C )} \reddC \term{\pRes{x\tBfr{\vec{m}}y} ( \tCtx{\hat{F}}[(M,y)] \prl C )}$.
            W.l.o.g., assume $\term{C}$ is a child thread.
            By \Cref{l:LASTn:LASTn:transCtxs}, for any $\term{L}$, $\encc{z}{\tCtx{F}[L]} = \evalCtx[\big]{E}[\encc{z'}{L}]$ for some $\evalCtx{E},z'$ ($\ast_1$).
            Moreover, since $\tCtx{\hat{F}}$ does not have its hole under an explicit substitution, it does not capture any free variables of $M$; hence, by \Cref{l:LASTn:LASTn:transCtxs}, $\evalCtx{E}$ does not capture any free names of $\encc{u}{M}$ for any~$u$~($\ast_2$).
            The thesis holds as follows:
            \begin{align*}
                & \encc{z}{\pRes{x\tBfr{\vec{m},M}y} ( \tCtx{\hat{F}}[\tRecv y] \prl C )}
                \\ \displaybreak[1]
                &= \pRes{ax} \pRes{by} ( \encc{a\>b}{\tBfr{\vec{m},M}} \| \encc{z}{\tCtx{\hat{F}}[\tRecv y]} \| \pRes{\_\_} \encc{\_}{C} )
                \\ \displaybreak[1]
                &= \pRes{ax} \pRes{by} ( \encc{a\>b}{\tBfr{\vec{m},M}} \| \evalCtx[\big]{E}[\encc{z'}{\tRecv y}] \| \pRes{\_\_} \encc{\_}{C} )
                \tag{$\ast_1$}
                \\ \displaybreak[1]
                &= \begin{array}[t]{@{}l@{}}
                    \pRes{ax} \pRes{by} (
                    \\ ~~
                    \pRes{c_1d_1} \pRes{e_1f_1} ( \pIn c_1(\_,g_1) ; \encc{g_1}{M} \| \pIn b(\_,h_1) ; \pOut h_1[d_1,e_1] \| \encc{a\>f_1}{\tBfr{\vec{m}}} )
                    \\ ~~
                    {} \| \evalCtx[\big]{E}[\pRes{a_2b_2} \big( \pOut y[\_,a_2] \| \pIn b_2(c_2,d_2) ; \pRes{e_2,f_2} ( \pOut z'[c_2,e_2] \| \pIn f_2(\_,g_2) ; \pOut d_2[\_,g_2] ) \big)]
                    \\ ~~
                    {} \| \pRes{\_\_} \encc{\_}{C} )
                \end{array}
                \\ \displaybreak[1]
                &\redd \begin{array}[t]{@{}l@{}}
                    \pRes{ax} \pRes{a_2b_2} (
                    \\ ~~
                    \pRes{c_1d_1} \pRes{e_1f_1} ( \pIn c_1(\_,g_1) ; \encc{g_1}{M} \| \pOut a_2[d_1,e_1] \| \encc{a\>f_1}{\tBfr{\vec{m}}} )
                    \\ ~~
                    {} \| \evalCtx[\big]{E}[\pIn b_2(c_2,d_2) ; \pRes{e_2,f_2} ( \pOut z'[c_2,e_2] \| \pIn f_2(\_,g_2) ; \pOut d_2[\_,g_2] )]
                    \\ ~~
                    {} \| \pRes{\_\_} \encc{\_}{C} )
                \end{array}
                \\ \displaybreak[1]
                &\redd \begin{array}[t]{@{}l@{}}
                    \pRes{ax} \pRes{d_1c_1} \pRes{f_1e_1} (
                    \\ ~~
                    \pIn c_1(\_,g_1) ; \encc{g_1}{M} \| \encc{a\>f_1}{\tBfr{\vec{m}}}
                    \\ ~~
                    {} \| \evalCtx[\big]{E}[\pRes{e_2,f_2} ( \pOut z'[d_1,e_2] \| \pIn f_2(\_,g_2) ; \pOut e_1[\_,g_2] )]
                    \\ ~~
                    {} \| \pRes{\_\_} \encc{\_}{C} )
                \end{array}
                \\ \displaybreak[1]
                &\equiv \begin{array}[t]{@{}l@{}}
                    \pRes{ax} \pRes{f_1e_1} (
                    \\ ~~
                    \encc{a\>f_1}{\tBfr{\vec{m}}}
                    \\ ~~
                    {} \| \evalCtx[\big]{E}[ \pRes{d_1c_1} \pRes{e_2,f_2} ( \pOut z'[d_1,e_2] \| \pIn c_1(\_,g_1) ; \encc{g_1}{M} \| \pIn f_2(\_,g_2) ; \pOut e_1[\_,g_2] )]
                    \\ ~~
                    {} \| \pRes{\_\_} \encc{\_}{C} )
                \end{array}
                \tag{$\ast_2$}
                \\ \displaybreak[1]
                &\equiv \begin{array}[t]{@{}l@{}}
                    \pRes{ax} \pRes{by} (
                    \\ ~~
                    \encc{a\>b}{\tBfr{\vec{m}}}
                    \\ ~~
                    {} \| \evalCtx[\big]{E}[ \pRes{d_1c_1} \pRes{e_2,f_2} ( \pOut z'[d_1,e_2] \| \pIn c_1(\_,g_1) ; \encc{g_1}{M} \| \pIn f_2(\_,g_2) ; \pOut y[\_,g_2] )]
                    \\ ~~
                    {} \| \pRes{\_\_} \encc{\_}{C} )
                \end{array}
                \\ \displaybreak[1]
                &= \pRes{ax} \pRes{by} ( \encc{a\>b}{\tBfr{\vec{m}}} \| \evalCtx[\big]{E}[\encc{z'}{(M,y)}] \| \pRes{\_\_} \encc{\_}{C} )
                \\ \displaybreak[1]
                &= \pRes{ax} \pRes{by} ( \encc{a\>b}{\tBfr{\vec{m}}} \| \encc{z}{\tCtx{\hat{F}}[(M,y)]} \| \pRes{\_\_} \encc{\_}{C} )
                \tag{$\ast_1$}
                \\
                &= \encc{z}{\pRes{x\tBfr{\vec{m}}y} ( \tCtx{\hat{F}}[(M,y)] \prl C )}
            \end{align*}

        \item
            Rule~\ruleLabel{red-select}:
            $\term{\pRes{x\tBfr{\vec{m}}y} ( \tCtx{F}[\tSel \ell \, x] \prl C )} \reddC \term{\pRes{x\tBfr{\ell,\vec{m}}y} ( \tCtx{F}[x] \prl C )}$.
            W.l.o.g., assume $\term{C}$ is a child thread.
            By \Cref{l:LASTn:LASTn:transCtxs}, for any $\term{L}$, $\encc{z}{\tCtx{F}[L]} = \evalCtx[\big]{E_1}[\encc{z'}{L}]$ for some $\evalCtx{E_1},z'$ ($\ast_1$).
            By inversion of typing, $\type{\Gamma} \vdashB \term{\tBfr{\vec{m}}} : \type{S_1}>\type{S}$ where $\type{S_1} = \type{\oplus \{ i : S_i \} \cup \{ \ell : S_2 \}}$~($\ast_2$).
            By \Cref{l:LASTn:LASTn:transBuf}, $\encc{a\>b}{\type{\Gamma} \vdashB \term{\tBfr{\vec{m}}} : \type{S_1}>\type{S}} = \evalCtx[\big]{E_2}[\encc{a\>c}{\type{\emptyset} \vdashB \term{\tBfr{\epsi}} : \type{S_1}>\type{S_1}}]$ ($\ast_3$) and $\encc{a\>b}{\type{\Gamma} \vdashB \term{\tBfr{\ell,\vec{m}}} : \type{S_1}>\type{S}} = \evalCtx[\big]{E_2}[\encc{a\>c}{\type{\emptyset} \vdashB \term{\tBfr{\ell}} : \type{S_1}>\type{S_1}}]$ ($\ast_4$) for some $\evalCtx{E_2},c$.
            Below, we omit types from translations of buffers.
            The thesis holds as follows:
            \begin{align*}
                & \encc{z}{\pRes{x\tBfr{\vec{m}}y} ( \tCtx{F}[\tSel \ell \, x] \prl C )}
                \\ \displaybreak[1]
                &= \pRes{ax} \pRes{by} ( \encc{a\>b}{\tBfr{\vec{m}}} \| \encc{z}{\tCtx{F}[\tSel \ell \, x]} \| \pRes{\_\_} \encc{\_}{C} )
                \\ \displaybreak[1]
                &= \pRes{ax} \pRes{by} ( \evalCtx[\big]{E_2}[\encc{a\>c}{\tBfr{\epsi}}] \| \evalCtx[\big]{E_1}[\encc{z'}{\tSel \ell \, x}] \| \pRes{\_\_} \encc{\_}{C} )
                \tag{$\ast_1$,$\ast_3$}
                \\ \displaybreak[1]
                &= \begin{array}[t]{@{}l@{}}
                    \pRes{ax} \pRes{by} (
                    \\ ~~
                    \evalCtx[\big]{E_2}[\pIn a(\_,c_1) ; \pRes{d_1e_1} \big( \pOut c_1[\_,d_1] \| \pBra e_1(f_1) > \{ i : \ldots \}_{i \in I} \cup \{ \ell : \pRes{g_1h_1} \left( \begin{array}{@{}l@{}}
                        \pIn c(\_,k_1) ; \pSel k_1[g_1] < \ell \\
                        {} \| \encc{f_1\>h_1}{\tBfr{\epsi}} 
                    \end{array} \right) \} \big)]
                    \\ ~~
                    {} \| \evalCtx[\big]{E_1}[\pRes{a_2b_2} \big( \pOut x[\_,a_2] \| \pIn b_2(\_,c_2) ; \pRes{d_2e_2} ( \pSel c_2[d_2] < \ell \| \pOut e_2[\_,z'] ) \big)]
                    \\ ~~
                    {} \| \pRes{\_\_} \encc{\_}{C} )
                \end{array}
                \tag{$\ast_2$}
                \\ \displaybreak[1]
                &\redd \begin{array}[t]{@{}l@{}}
                    \pRes{a_2b_2} \pRes{by} (
                    \\ ~~
                    \evalCtx[\big]{E_2}[\pRes{d_1e_1} \big( \pOut a_2[\_,d_1] \| \pBra e_1(f_1) > \{ i : \ldots \}_{i \in I} \cup \{ \ell : \pRes{g_1h_1} \left( \begin{array}{@{}l@{}}
                        \pIn c(\_,k_1) ; \pSel k_1[g_1] < \ell \\
                        {} \| \encc{f_1\>h_1}{\tBfr{\epsi}} 
                    \end{array} \right) \} \big)]
                    \\ ~~
                    {} \| \evalCtx[\big]{E_1}[\pIn b_2(\_,c_2) ; \pRes{d_2e_2} ( \pSel c_2[d_2] < \ell \| \pOut e_2[\_,z'] )]
                    \\ ~~
                    {} \| \pRes{\_\_} \encc{\_}{C} )
                \end{array}
                \\ \displaybreak[1]
                &\redd \begin{array}[t]{@{}l@{}}
                    \pRes{a_2b_2} \pRes{by} (
                    \\ ~~
                    \evalCtx[\big]{E_2}[\pRes{d_1e_1} \big( \pOut a_2[\_,d_1] \| \pBra e_1(f_1) > \{ i : \ldots \}_{i \in I} \cup \{ \ell : \pRes{g_1h_1} \left( \begin{array}{@{}l@{}}
                        \pIn c(\_,k_1) ; \pSel k_1[g_1] < \ell \\
                        {} \| \encc{f_1\>h_1}{\tBfr{\epsi}} 
                    \end{array} \right) \} \big)]
                    \\ ~~
                    {} \| \evalCtx[\big]{E_1}[\pIn b_2(\_,c_2) ; \pRes{d_2e_2} ( \pSel c_2[d_2] < \ell \| \pOut e_2[\_,z'] )]
                    \\ ~~
                    {} \| \pRes{\_\_} \encc{\_}{C} )
                \end{array}
                \\ \displaybreak[1]
                &\redd \begin{array}[t]{@{}l@{}}
                    \pRes{d_2e_2} \pRes{by} (
                    \\ ~~
                    \evalCtx[\big]{E_2}[\pRes{g_1h_1} ( \pIn c(\_,k_1) ; \pSel k_1[g_1] < \ell \| \encc{d_2\>h_1}{\tBfr{\epsi}} )]
                    \\ ~~
                    {} \| \evalCtx[\big]{E_1}[\pOut e_2[\_,z']]
                    \\ ~~
                    {} \| \pRes{\_\_} \encc{\_}{C} )
                \end{array}
                \\ \displaybreak[1]
                &\equiv \begin{array}[t]{@{}l@{}}
                    \pRes{ax} \pRes{by} (
                    \\ ~~
                    \evalCtx[\big]{E_2}[\pRes{g_1h_1} ( \pIn c(\_,k_1) ; \pSel k_1[g_1] < \ell \| \encc{a\>h_1}{\tBfr{\epsi}} )]
                    \\ ~~
                    {} \| \evalCtx[\big]{E_1}[\pOut x[\_,z']]
                    \\ ~~
                    {} \| \pRes{\_\_} \encc{\_}{C} )
                \end{array}
                \\ \displaybreak[1]
                &= \pRes{ax} \pRes{by} ( \evalCtx[\big]{E_2}[\encc{a\>c}{\tBfr{\ell}} )] \| \evalCtx[\big]{E_1}[\encc{z'}{x} )] \| \pRes{\_\_} \encc{\_}{C} )
                \\ \displaybreak[1]
                &= \pRes{ax} \pRes{by} ( \encc{a\>b}{\tBfr{\ell,\vec{m}}} )] \| \encc{z}{\tCtx{F}[x]} \| \pRes{\_\_} \encc{\_}{C} )
                \tag{$\ast_1$,$\ast_4$}
                \\ \displaybreak[1]
                &= \encc{z}{\pRes{x\tBfr{\ell,\vec{m}}y} ( \tCtx{F}[x] \prl C )}
            \end{align*}

        \item
            Rule~\ruleLabel{red-case}:
            $j \in I$ implies
            \[
                \term{\pRes{x\tBfr{\vec{m},j}y} ( \tCtx{F}[\tCase y \tOf \{ i : M_i \}_{i \in I}] \prl C )} \reddC \term{\pRes{x\tBfr{\vec{m}}y} ( \tCtx{F}[M_j\ y] \prl C )}.
            \]
            W.l.o.g., assume $\term{C}$ is a child thread.
            By \Cref{l:LASTn:LASTn:transCtxs}, for any $\term{L}$, $\encc{z}{\tCtx{F}[L]} = \evalCtx[\big]{E}[\encc{z'}{L}]$ for some $\evalCtx{E},z'$ ($\ast$).
            Assume the condition.
            The thesis holds as follows:
            \begin{align*}
                & \encc{z}{\pRes{x\tBfr{\vec{m},j}y} ( \tCtx{F}[\tCase y \tOf \{ i : M_i \}_{i \in I}] \prl C )}
                \\ \displaybreak[1]
                &= \pRes{ax} \pRes{by} ( \encc{a\>b}{\tBfr{\vec{m},j}} \| \encc{z}{\tCtx{F}[\tCase y \tOf \{ i : M_i \}_{i \in I}]} \| \pRes{\_\_} \encc{\_}{C} )
                \\ \displaybreak[1]
                &= \pRes{ax} \pRes{by} ( \encc{a\>b}{\tBfr{\vec{m},j}} \| \evalCtx[\big]{E}[\encc{z'}{\tCase y \tOf \{ i : M_i \}_{i \in I}}] \| \pRes{\_\_} \encc{\_}{C} )
                \tag{$\ast$}
                \\ \displaybreak[1]
                &= \begin{array}[t]{@{}l@{}}
                    \pRes{ax} \pRes{by} (
                    \\ ~~
                    \pRes{c_1d_1} ( \pIn b(\_,e_1) ; \pSel e_1[c_1] < j \| \encc{a\>d_1}{\tBfr{\vec{m}}} )
                    \\ ~~
                    {} \| \evalCtx[\big]{E}[\pRes{a_2b_2} ( \pOut y[\_,a_2] \| \pBra b_2(c_2) > \{ i : \encc{z'}{M_i\ c_2} \}_{i \in I} )]
                    \\ ~~
                    {} \| \pRes{\_\_} \encc{\_}{C} )
                \end{array}
                \\ \displaybreak[1]
                &\redd \begin{array}[t]{@{}l@{}}
                    \pRes{ax} \pRes{a_2b_2} (
                    \\ ~~
                    \pRes{c_1d_1} (\pSel a_2[c_1] < j \| \encc{a\>d_1}{\tBfr{\vec{m}}} )
                    \\ ~~
                    {} \| \evalCtx[\big]{E}[\pBra b_2(c_2) > \{ i : \encc{z'}{M_i\ c_2} \}_{i \in I}]
                    \\ ~~
                    {} \| \pRes{\_\_} \encc{\_}{C} )
                \end{array}
                \\ \displaybreak[1]
                &\redd \pRes{ax} \pRes{d_1c_1} ( \encc{a\>d_1}{\tBfr{\vec{m}}} \| \evalCtx[\big]{E}[\encc{z'}{M_j\ c_2} \{ c_1/c_2 \}] \| \pRes{\_\_} \encc{\_}{C} )
                \\ \displaybreak[1]
                &= \pRes{ax} \pRes{d_1c_1} ( \encc{a\>d_1}{\tBfr{\vec{m}}} \| \evalCtx[\big]{E}[\encc{z'}{M_j\ c_1}] \| \pRes{\_\_} \encc{\_}{C} )
                \tag{\Cref{l:LASTn:LASTn:transCtxs}}
                \\ \displaybreak[1]
                &\equiv \pRes{ax} \pRes{by} ( \encc{a\>b}{\tBfr{\vec{m}}} \| \evalCtx[\big]{E}[\encc{z'}{M_j\ y}] \| \pRes{\_\_} \encc{\_}{C} )
                \tag{\Cref{l:LASTn:LASTn:transCtxs}}
                \\ \displaybreak[1]
                &= \pRes{ax} \pRes{by} ( \encc{a\>b}{\tBfr{\vec{m}}} \| \encc{z}{\tCtx{F}[M_j\ y]} \| \pRes{\_\_} \encc{\_}{C} )
                \tag{$\ast$}
                \\ \displaybreak[1]
                &= \encc{z}{\pRes{x\tBfr{\vec{m}}y} ( \tCtx{F}[M_j\ y] \prl C )}
            \end{align*}

        \item
            Rule~\ruleLabel{red-close}:
            $\term{\pRes{x\tBfr{\vec{m}}y} ( \tCtx{F}[\tClose x ; M] \prl C )} \reddC \term{\pRes{\tNil\tBfr{\vec{m}}y} ( \tCtx{F}[M] \prl C )}$.
            W.l.o.g., assume $\term{C}$ is a child thread.
            By \Cref{l:LASTn:LASTn:transCtxs}, for any $\term{L}$, $\encc{z}{\tCtx{F}[L]} = \evalCtx[\big]{E_1}[\encc{z'}{L}]$ for some $\evalCtx{E_1},z'$ ($\ast_1$).
            The analysis depends on whether $\term{y} = \term{\tNil}$; w.l.o.g., assume not.
            By inversion of typing, $\type{\Gamma} \vdashB \term{\tBfr{\vec{m}}} : \type{\tEnd}>\type{S}$ where $\type{S} \neq \type{\tNil}$ ($\ast_2$).
            By \Cref{l:LASTn:LASTn:transBuf}, $\encc{a\>b}{\type{\Gamma} \vdashB \term{\tBfr{\vec{m}}} : \type{\tEnd}>\type{S}} = \evalCtx[\big]{E_2}[\encc{a\>c}{\type{\emptyset} \vdashB \term{\tBfr{\epsi}} : \type{\tEnd}>\type{\tEnd}}]$ ($\ast_3$) and $\encc{a\>b}{\type{\Gamma} \vdashB \term{\tBfr{\vec{m}}} : \type{\tNil}>\type{S}} = \evalCtx[\big]{E_2}[\encc{a\>c}{\type{\emptyset} \vdashB \term{\tBfr{\epsi}} : \type{\tNil}>\type{\tEnd}}]$ ($\ast_4$) for some $\evalCtx{E_2},c$.
            Below, we omit types from translations of buffers.
            The thesis holds as follows:
            \begin{align*}
                & \encc{z}{\pRes{x\tBfr{\vec{m}}y} ( \tCtx{F}[\tClose x ; M] \prl C )}
                \\ \displaybreak[1]
                &= \pRes{ax} \pRes{by} ( \encc{a\>b}{\tBfr{\vec{m}}} \| \encc{z}{\tCtx{F}[\tClose x ; M]} \| \pRes{\_\_} \encc{\_}{C} )
                \\ \displaybreak[1]
                &= \pRes{ax} \pRes{by} ( \evalCtx[\big]{E_2}[\encc{a\>c}{\tBfr{\epsi}}] \| \evalCtx[\big]{E_1}[\encc{z'}{\tClose x ; M}] \| \pRes{\_\_} \encc{\_}{C} )
                \tag{$\ast_1$,$\ast_3$}
                \\ \displaybreak[1]
                &= \begin{array}[t]{@{}l@{}}
                    \pRes{ax} \pRes{by} (
                    \\ ~~
                    \evalCtx[\big]{E_2}[\pIn a(\_,c_1) ; \pOut c_1[\_,\_] \| \pIn b(\_,d_1) ; \pOut d_1[\_,\_]]
                    \\ ~~
                    {} \| \evalCtx[\big]{E_1}[\pRes{a_2b_2} ( \pOut x[\_,a_2] \| \pIn b_2(\_,\_) ; \encc{z'}{M} )]
                    \\ ~~
                    {} \| \pRes{\_\_} \encc{\_}{C} )
                \end{array}
                \tag{$\ast_2$}
                \\ \displaybreak[1]
                &\redd \begin{array}[t]{@{}l@{}}
                    \pRes{a_2b_2} \pRes{by} (
                    \\ ~~
                    \evalCtx[\big]{E_2}[\pOut a_2[\_,\_] \| \pIn b(\_,d_1) ; \pOut d_1[\_,\_]]
                    \\ ~~
                    {} \| \evalCtx[\big]{E_1}[\pIn b_2(\_,\_) ; \encc{z'}{M}]
                    \\ ~~
                    {} \| \pRes{\_\_} \encc{\_}{C} )
                \end{array}
                \\ \displaybreak[1]
                &\redd \pRes{by} ( \evalCtx[\big]{E_2}[\pIn b(\_,d_1) ; \pOut d_1[\_,\_]] \| \evalCtx[\big]{E_1}[\encc{z'}{M}] \| \pRes{\_\_} \encc{\_}{C} )
                \\ \displaybreak[1]
                &\equiv \pRes{a\_} \pRes{by} ( \evalCtx[\big]{E_2}[\pIn b(\_,d_1) ; \pOut d_1[\_,\_]] \| \evalCtx[\big]{E_1}[\encc{z'}{M}] \| \pRes{\_\_} \encc{\_}{C} )
                \\ \displaybreak[1]
                &= \pRes{a\_} \pRes{by} ( \evalCtx[\big]{E_2}[\encc{a\>c}{\tBfr{\epsi}}] \| \evalCtx[\big]{E_1}[\encc{z'}{M}] \| \pRes{\_\_} \encc{\_}{C} )
                \\ \displaybreak[1]
                &= \pRes{a\_} \pRes{by} ( \encc{a\>b}{\tBfr{\vec{m}}} \| \encc{z}{\tCtx{F}[M]} \| \pRes{\_\_} \encc{\_}{C} )
                \tag{$\ast_1$,$\ast_4$}
                \\ \displaybreak[1]
                &= \encc{z}{\pRes{\tNil\tBfr{\vec{m}}y} ( \tCtx{F}[M] \| C )}
            \end{align*}

        \item
            Rule~\ruleLabel{red-par-nil}:
            $\term{C \prl \tChild \, ()} \reddC \term{C}$.
            The thesis holds as follows:
            \begin{align*}
                \encc{z}{C \prl \tChild \, ()} &= \encc{z}{C} \| \pRes{\_\_} \0
                \\
                &\equiv \encc{z}{C}
            \end{align*}

        \item
            Rule~\ruleLabel{red-res-nil}:
            $\term{\pRes{\tNil\tBfr{\epsi}\tNil} C} \reddC \term{C}$.
            The thesis holds as follows:
            \begin{align*}
                \encc{z}{\pRes{\tNil\tBfr{\epsi}\tNil} C} &= \pRes{a\_} \pRes{b\_} ( \0 \| \encc{z}{C} )
                \\
                &\equiv \encc{z}{C}
            \end{align*}

        \item
            Rule~\ruleLabel{red-lift-C}:
            $\term{C} \reddC \term{C'}$ implies $\term{\tCtx{G}[C]} \reddC \term{\tCtx{G}[C']}$.
            By \Cref{l:LASTn:LASTn:transCtxs}, for any~$\term{D}$, $\encc{z}{\tCtx{G}[D]} = \evalCtx[\big]{E}[\encc{z'}{D}]$ for some $\evalCtx{E},z'$ ($\ast_1$).
            Assume the condition.
            By the IH, \mbox{$\encc{z'}{C} \redd* \encc{z'}{C'}$} ($\ast_2$).
            The thesis holds as follows:
            \begin{align*}
                \encc{z}{\tCtx{G}[C]} &= \evalCtx[\big]{E}[\encc{z'}{C}]
                \tag{$\ast_1$}
                \\ \displaybreak[1]
                &\redd* \evalCtx[\big]{E}[\encc{z'}{C'}]
                \tag{$\ast_2$}
                \\ \displaybreak[1]
                &= \encc{z}{\tCtx{G}[C']}
                \tag{$\ast_1$}
            \end{align*}

        \item
            Rule~\ruleLabel{red-lift-M}:
            $\term{M} \reddM \term{M'}$ implies $\term{\tCtx{F}[M]} \reddC \term{\tCtx{F}[M']}$.
            By \Cref{l:LASTn:LASTn:transCtxs}, for any $\term{N}$, $\encc{z}{\tCtx{F}[N]} = \evalCtx[\big]{E}[\encc{z'}{N}]$ for some $\evalCtx{E},z'$ ($\ast_1$).
            Assume the condition.
            By \Cref{t:LASTn:LASTn:transTermRed}, $\encc{z'}{M} \redd* \encc{z'}{M'}$ ($\ast_2$).
            The thesis holds as follows:
            \begin{align*}
                \encc{z}{\tCtx{F}[M]} &= \evalCtx[\big]{E}[\encc{z'}{M}]
                \tag{$\ast_1$}
                \\ \displaybreak[1]
                &\redd* \evalCtx[\big]{E}[\encc{z'}{M'}]
                \tag{$\ast_2$}
                \\ \displaybreak[1]
                &= \encc{z}{\tCtx{F}[M']}
                \tag{$\ast_1$}
            \end{align*}

        \item
            Rule~\ruleLabel{red-conf-lift-sc}:
            $\term{C} \equivC \term{C'}$, $\term{C'} \reddC \term{D'}$, and $\term{D'} \equivC \term{D}$ imply $\term{C} \reddC \term{D}$.
            Assume the conditions.
            By \Cref{t:LASTn:LASTn:transConfSc}, $\encc{z}{C} \equiv \encc{z}{C'}$ ($\ast_1$) and $\encc{z}{D'} \equiv \encc{z}{D}$ ($\ast_2$).
            By the IH, $\encc{z}{C'} \redd* \encc{z}{D'}$ ($\ast_3$).
            The thesis holds as follows:
            \begin{align*}
                \encc{z}{C} &\equiv \encc{z}{C'}
                \tag{$\ast_1$}
                \\
                &\redd* \encc{z}{D'}
                \tag{$\ast_3$}
                \\
                &\equiv \encc{z}{D}
                \tag{$\ast_2$}
                \\
                \tag*{\qedhere}
            \end{align*}

    \end{itemize}
\end{proof}


\subsubsection{Soundness}
\label{as:LASTn:oc:sound}


\tLASTnTransSound*

\begin{proof}
    \proofLASTnTransSoundIntro
    \begin{itemize}

        \item
            Case $\term{C} = \term{\phi \, M}$.
            By construction, we can identify a maximal context $\tCtx{R}$ and a term~$\term{M_0}$ such that $\term{M} = \tCtx{R}[M_0]$ and the observed reduction $\encc{z}{\phi \, (\tCtx{R}[M_0])} \redd Q_0$ originates from the translation of $M_0$ directly (i.e., not from inside an evaluation context in the translation of $\term{M_0}$ or from interaction with the translation of $\tCtx{R}$).

            We detail every case for $\term{M_0}$, though not all cases may be applicable to show a reduction.
            In each case, we rely on \Cref{l:LASTn:LASTn:transCtxs} to work with $\evalCtx{E},z'$ such that \mbox{$\encc{z}{\phi \, (\tCtx{R}[M_0])} = \encc{z}{\tCtx{R}[M_0]} = \evalCtx[\big]{E}[\encc{z'}{M_0}]$}.
            Also, in many cases, subterms that partake in the reduction may appear in sequences of explicit substitutions; since structural congruence can always extrude the scope of explicit substitutions, they can, w.l.o.g., be factored out of the proofs below.
            \begin{itemize}

                \item
                    Case $\term{M_0} = \term{x}$.
                    We have $\encc{z'}{x} = \pOut x[\_,z']$, so no reduction is possible.

                \item
                    Case $\term{M_0} = \term{()}$.
                    We have $\encc{z'}{()} = \0$, so no reduction is possible.

                \item
                    Case $\term{M_0} = \term{\lam x . M_1}$.
                    We have $\encc{z'}{\lam x . M_1} = \pIn z'(x,a) ; \encc{a}{M_1}$, so no reduction is possible.

                \item
                    Case $\term{M_0} = \term{M_1\ M_2}$.
                    We have
                    \[
                        \encc{z'}{M_1\ M_2} = \pRes{a_1b_1} \pRes{c_1d_1} ( \encc{a_1}{M_1} \| \pOut b_1[c_1,z'] \| \pIn d_1(\_,e_1) ; \encc{e_1}{M_2} ).
                    \]
                    The reduction can only originate from a synchronization between the send on~$b_1$ and a receive on $a_1$ in $\encc{a_1}{M_1}$.
                    By well-typedness, $\term{M_1}$ must be of type $\type{T_2 \lolli T_1}$ and $\term{M_2}$ of type $\type{T_2}$.
                    It must then be the case that $\term{M_1} = \term{\lam x . M_{1.1}}$: this is the only possibility for a receive on $a_1$ in $\encc{a_1}{M_1}$.

                    Let $\term{D_0} \deq \term{\phi \, (\tCtx{R}[M_{1.1} \tSub{ M_2/x })]}$.
                    We have $\term{C} \reddC \term{D_0}$.
                    Moreover:
                    \begin{align*}
                        \encc{z}{C}
                        &= \evalCtx[\big]{E}[\pRes{a_1b_1} \pRes{c_1d_1} ( \pIn a_1(x,a_2) ; \encc{a_2}{M_{1.1}} \| \pOut b_1[c_1,z'] \| \pIn d_1(\_,e_1) ; \encc{e_1}{M_2} )]
                        \\
                        &\redd \evalCtx[\big]{E}[\pRes{xd_1} ( \encc{z'}{M_{1.1}} \| \pIn d_1(\_,e_1) ; \encc{e_1}{M_2} )]
                        \\
                        &= \encc{z}{D_0}
                    \end{align*}

                \item
                    Case $\term{M_0} = \term{\tNew}$.
                    We have
                    \[
                        \encc{z'}{\tNew} = \pRes{a_1b_1} \big( \pOut a_1[\_,z'] \| \pIn b_1(\_,c_1) ; \pRes{d_1x} \pRes{e_1y} ( \encc{d_1\>e_1}{\tBfr{\epsi}} \| \encc{c_1}{(x,y)} ) \big).
                    \]
                    The reduction can only originate from a synchronization between the send on~$a_1$ and the receive on $b_1$.
                    Let $\term{D_0} \deq \term{\pRes{x\tBfr{\epsi}y} ( \phi \, (\tCtx{R}[(x,y)]) )}$.
                    We have $\term{C} \reddC \term{D_0}$.
                    Moreover:
                    \begin{align*}
                        \encc{z}{C}
                        &= \evalCtx[\big]{E}[\pRes{a_1b_1} \big( \pOut a_1[\_,z'] \| \pIn b_1(\_,c_1) ; \pRes{d_1x} \pRes{e_1y} ( \encc{d_1\>e_1}{\tBfr{\epsi}} \| \encc{c_1}{(x,y)} ) \big)]
                        \\
                        &\redd \pRes{d_1x} \pRes{e_1y} ( \encc{d_1\>e_1}{\tBfr{\epsi}} \| \evalCtx[\big]{E}[\encc{z'}{(x,y)}] )
                        \\
                        &= \encc{z}{D_0}
                    \end{align*}

                \item
                    Case $\term{M_0} = \term{\tSpawn M_1 ; M_2}$.
                    We have
                    \[
                        \encc{z'}{\tSpawn M_1 ; M_2} = \pRes{a_1} \big( \pOut a_1[\_,z'] \| \pIn b_1(\_,c_1) ; ( \pRes{\_\_} \encc{\_}{M_1} \| \encc{c_1}{M_2} ) \big).
                    \]
                    The reduction can only originate from a synchronization between the send on $a_1$ and the receive on $b_1$.
                    Let $\term{D_0} \deq \term{\phi \, (\tCtx{R}[M_2]) \prl \tChild \, M_1}$.
                    We have $\term{C} \reddC \term{D_0}$.
                    Moreover:
                    \begin{align*}
                        \encc{z}{C}
                        &= \evalCtx[\big]{E}[\pRes{a_1} \big( \pOut a_1[\_,z'] \| \pIn b_1(\_,c_1) ; ( \pRes{\_\_} \encc{\_}{M_1} \| \encc{c_1}{M_2} ) \big)]
                        \\
                        &\redd \evalCtx[\big]{E}[\encc{z'}{M_2}] \| \pRes{\_\_} \encc{\_}{M_1}
                        \\
                        &= \encc{z}{D_0}
                    \end{align*}

                \item
                    Case $\term{M_0} = \term{(M_1,M_2)}$.
                    We have
                    \[
                        \encc{z'}{(M_1,M_2)} = \pRes{a_1b_1} \pRes{c_1d_1} ( \pOut z'[a_1,c_1] \| \pIn b_1(\_,e_1) ; \encc{e_1}{M_1} \| \pIn d_1(\_,f_1) ; \encc{f_1}{M_2} ),
                    \]
                    so no reduction is possible.

                \item
                    Case $\term{M_0} = \term{\tLet (x,y) = M_1 \tIn M_2}$.
                    We have
                    \[
                        \encc{z'}{\tLet (x,y) = M_1 \tIn M_2} = \pRes{a_1b_1} ( \pIn a_1(x,y) ; \encc{z'}{M_2} \| \encc{b_1}{M_1} ).
                    \]
                    The reduction can only originate from a synchronization between the receive on~$a_1$ and a send on $b_1$ in $\encc{b_1}{M_1}$.
                    By well-typedness, $\term{M_1}$ must be of type~$\type{T_{1.1} \times T_{1.2}}$.
                    It must then be the case that $\term{M_1} = \term{(M_{1.1},M_{1.2})}$: this is the only possibility for a send on $b_1$ in $\encc{b_1}{M_1}$.

                    Let $\term{D_0} \deq \term{\phi \, (\tCtx{R}[M_1 \tSub{ M_{1.1}/x,M_{1.2}/y }])}$.
                    We have $\term{C} \reddC \term{D_0}$.
                    Moreover:
                    \begin{align*}
                        & \encc{z}{C}
                        \\
                        &= \evalCtx[\big]{E}[\pRes{a_1b_1} ( \pIn a_1(x,y) ; \encc{z'}{M_2} \| \pRes{a_2b_2} \pRes{c_2d_2} \left( \begin{array}{@{}l@{}}
                            \pOut b_1[a_2,c_2] \| \pIn b_2(\_,e_2) ; \encc{e_2}{M_{1.1}}
                            \\
                            {} \| \pIn d_2(\_,f_2) ; \encc{f_2}{M_{1.2}}
                        \end{array} \right) )]
                        \\
                        &\redd \evalCtx[\big]{E}[\pRes{yd_2} \big( \pRes{xb_2} ( \encc{z'}{M_2} \| \pIn b_2(\_,e_2) ; \encc{e_2}{M_{1.1}} ) \| \pIn d_2(\_,f_2) ; \encc{f_2}{M_{1.2}} \big)]
                        \\
                        &= \encc{z}{D_0}
                    \end{align*}

                \item
                    Case $\term{M_0} = \term{\tSend M_1 \, M_2}$.
                    We have
                    \[
                        \encc{z'}{\tSend M_1 \, M_2} = \pRes{a_1b_1} \pRes{c_1d_1} \left(
                            \begin{array}{@{}l@{}}
                                \pIn a_1(\_,e_1) ; \encc{e_1}{M_1} \| \encc{c_1}{M_2}
                                \\
                                {} \| \pIn d_1(\_,f_1) ; \pRes{g_1h_1} ( \pOut f_1[b_1,g_1] \| \pOut h_1[\_,z'] )
                            \end{array}
                        \right).
                    \]
                    The reduction can only originate from a synchronization between the receive on~$d_1$ and a send on $c_1$ in $\encc{c_1}{M_2}$.
                    By well-typedness, $\term{M_2}$ must be of type $\type{{!}T_1 . S_2}$.
                    No reduction is possible: no $\term{M_2}$ can satisy these conditions.

                \item
                    Case $\term{M_0} = \term{\tRecv M_1}$.
                    We have
                    \[
                        \encc{z'}{\tRecv M_1} = \pRes{a_1b_1} \big( \encc{a_1}{M_1} \| \pIn b_1(c_1,d_1) ; \pRes{e_1f_1} ( \pOut z'[c_1,e_1] \| \pIn f_1(\_,g_1) ; \pOut d_1[\_,g_1] ) \big).
                    \]
                    The reduction can only originate from a synchronization between the receive on~$b_1$ and a send on $a_1$ in $\encc{a_1}{M_1}$.
                    By well-typedness, $\term{M_1}$ must be of type $\type{{?}T_1.S_1}$.
                    No reduction is possible: no $\term{M_1}$ can satisfy these conditions.

                \item
                    Case $\term{M_0} = \term{\tSel j \, M_1}$.
                    We have
                    \[
                        \encc{z'}{\tSel j \, M_1} = \pRes{a_1b_1} \big( \encc{a_1}{M_1} \| \pIn b_1(\_,c_1) ; \pRes{d_1e_1} ( \pSel c_1[d_1] < j \| \pOut e_1[\_,z'] ) \big).
                    \]
                    The reduction can only originate from a synchronization between the receive on~$b_1$ and a send on $a_1$ in $\encc{a_1}{M_1}$.
                    By well-typedness, $\term{M_1}$ must be of type~$\type{\oplus \{ i : S_1^i \}_{i \in I}}$ with $j \in I$.
                    No reduction is possible: no $\term{M_1}$ can satisfy these conditions.

                \item
                    Case $\term{M_0} = \term{\tCase M_1 \tOf \{ i : M_2^i \}_{i \in I}}$.
                    We have
                    \[
                        \encc{z'}{\tCase M_1 \tOf \{ i : M_2^i \}_{i \in I}} = \pRes{a_1b_1} ( \encc{a_1}{M_1} \| \pBra b_1(c_2) > \{ i : \encc{z'}{M_2^i\ c_2} \}_{i \in I} ).
                    \]
                    The reduction can only originate from a synchronization between the branch on~$b_1$ and a selection on $a_1$ in $\encc{a_1}{M_1}$.
                    By well-typedness, $M_1$ must be of type~$\type{\& \{ i : S_2^i \}_{i \in I}}$.
                    No reduction is possible: no $\term{M_1}$ can satisfy these conditions.

                \item
                    Case $\term{M_0} = \term{\tClose M_1 ; M_2}$.
                    We have
                    \[
                        \encc{z'}{\tClose M_1 ; M_2} = \pRes{a_1b_1} ( \encc{a_1}{M_1} \| \pIn b_1(\_,\_) ; \encc{z'}{M_2} ).
                    \]
                    The reduction can only originate from a synchronization between the receive on $b_1$ and a send on $a_1$ in $\encc{a_1}{M_1}$.
                    By well-typedness, $\term{M_1}$ must be of type $\type{\tEnd}$.
                    No reduction is possible: no $\term{M_1}$ can satisfy these conditions.

                \item
                    Case $\term{M_0} = \term{M_1 \tSub{ M_2/x }}$.
                    We have
                    \[
                        \encc{z'}{M_1 \tSub{ M_2/x }} = \pRes{xa_1} ( \encc{z'}{M_1} \| \pIn a_1(\_,b_1) ; \encc{b_1}{M_2} ).
                    \]
                    The reduction can only originate from a synchronization between the receive on $a_1$ and a send on $x$ in $\encc{z'}{M_1}$.
                    It must then be the case that $\term{M_1} = \tCtx{R_1}[x]$.
                    By \Cref{l:LASTn:LASTn:transCtxs}, $\encc{z'}{\tCtx{R_1}[x]} = \evalCtx{E_1}[\encc{z_1}{x}] = \evalCtx{E_1}[\pOut x[\_,z_1]]$.

                    Let $\term{D_0} \deq \term{\phi \, (\tCtx[\big]{R}[\tCtx{R_1}[M_2]])}$.
                    We have $\term{C} \equivC \term{\phi \, (\tCtx[\big]{R}[\tCtx{R_1}[x \tSub{ M_1/x }]])} \reddC \term{D_0}$.
                    Moreover:
                    \begin{align*}
                        \encc{z}{C}
                        &= \evalCtx[\big]{E}[\pRes{xa_1} ( \evalCtx{E_1}[\pOut x[\_,z_1]] \| \pIn a_1(\_,b_1) ; \encc{b_1}{M_2} )]
                        \\
                        &\redd \evalCtx[\big]{E}[\evalCtx{E_1}[\encc{z_1}{M_2}]]
                        \\
                        &= \encc{z}{D_0}
                    \end{align*}

            \end{itemize}

        \item
            Case $\term{C} = \term{C_1 \prl C_2}$.
            Assume, w.l.o.g., that $\term{C_2}$ is a child thread.
            We have
            \[
                \encc{z}{C} = \encc{z}{C_1} \| \pRes{\_\_} \encc{\_}{C_2}.
            \]
            The reduction may originate from $\encc{z}{C_1}$ or from $\encc{\_}{C_2}$; w.l.o.g., assume the former.

            We thus have $\encc{z}{C_1} \redd Q_1$.
            By \ih2, there are $\term{D_1},k_1 \geq 0$ such that $\term{C_1} \reddC \term{D_1}$ and \mbox{$\encc{z}{C_1} \redd Q_1 \redd^{k_1} \encc{z}{D_1}$}.
            Let $\term{D_0} \deq \term{D_1 \prl C_2}$.
            We have $\term{C} \reddC \term{D_0}$.
            Moreover:
            \begin{align*}
                \encc{z}{C}
                &= \encc{z}{C_1} \| \pRes{\_\_} \encc{\_}{C_2}
                \\
                &\redd^{k_1+1} \encc{z}{D_1} \| \pRes{\_\_} \encc{\_}{C_2}
                \\
                &= \encc{z}{D_0}
            \end{align*}

        \item
            Case \proofLASTnTransSoundBuf{}

        \item
            Case $\term{C} = \term{C_1 \tSub{ M/x }}$.
            We have
            \[
                \encc{z}{C_1 \tSub{ M/x }} = \pRes{xa_1} ( \encc{z}{C_1} \| \pIn a_1(\_,b_1) ; \encc{b_1}{M} ).
            \]
            The reduction may originate from (i)~$\encc{z}{C_1}$ or (ii)~a synchronization between the receive on $a_1$ and a send on $x$ in $\encc{z}{C_1}$.
            We detail both cases:
            \begin{itemize}

                \item[(i)]
                    The reduction originates from $\encc{z}{C_1}$.
                    We thus have $\encc{z}{C_1} \redd Q_1$.
                    By \ih2, there are $\term{D_1},k_1 \geq 0$ such that $\term{C_1} \reddC \term{D_1}$ and $\encc{z}{C_1} \redd Q_1 \redd^{k_1} \encc{z}{D_1}$.
                    Let $\term{D_0} \deq \term{D_1 \tSub{ M/x }}$.
                    We have $\term{C} \reddC \term{D_0}$.
                    Moreover:
                    \begin{align*}
                        \encc{z}{C}
                        &= \pRes{xa_1} ( \encc{z}{C_1} \| \pIn a_1(\_,b_1) ; \encc{b_1}{M} )
                        \\
                        &\redd^{k_1+1} \pRes{xa_1} ( \encc{z}{D_1} \| \pIn a_1(\_,b_1) ; \encc{b_1}{M} )
                        \\
                        &= \encc{z}{D_0}
                    \end{align*}

                \item[(ii)]
                    The reduction originates from a synchronization between the receive on $a_1$ and a send on $x$ in $\encc{z}{C_1}$.
                    It must then be the case that $\term{C_1} = \tCtx[\big]{G}[\tCtx{F}[x]]$.
                    By \Cref{l:LASTn:LASTn:transCtxs}, $\encc{z}{\tCtx[\big]{G}[\tCtx{F}[x]]} = \evalCtx{E}[\encc{z'}{x}] = \evalCtx{E}[\pOut x[\_,z']]$.

                    Let $\term{D_0} \deq \tCtx[\big]{G}[\tCtx{F}[M]]$.
                    We have $\term{C} \equivC \tCtx[\big]{G}[\tCtx{F}[x \tSub{ M/x }]] \reddC \term{D_0}$.
                    Moreover:
                    \begin{align*}
                        \encc{z}{C}
                        &= \pRes{xa_1} ( \evalCtx{E}[\pOut x[\_,z']] \| \pIn a_1(\_,b_1) ; \encc{b_1}{M} )
                        \\
                        &\redd \evalCtx{E}[\encc{z'}{M}]
                        \\
                        &= \encc{z}{D_0}
                        \tag*{\qedhere}
                    \end{align*}

            \end{itemize}

    \end{itemize}
\end{proof}



\maketitle

\begin{abstract}
    Protocols provide the unifying glue in concurrent and distributed software today; verifying that message-passing programs conform to such governing protocols is important but difficult.
    Static approaches based on multiparty session types (\MPST) use protocols as types to avoid protocol violations and deadlocks in programs.
    An elusive problem for \MPST is to ensure \emph{both} protocol conformance \emph{and} deadlock-freedom for implementations with interleaved and delegated protocols.

    We propose a decentralized analysis of multiparty protocols, specified as global types and implemented as interacting processes in an asynchronous $\pi$-calculus.
    Our solution rests upon two novel notions: \emph{router processes} and \emph{relative types}.
    While router processes use the global type to enable the composition of participant implementations in arbitrary process networks, relative types extract from the global type the intended interactions and dependencies between \emph{pairs} of participants.
    In our analysis, processes are typed using \APCP, a type system that ensures protocol conformance and deadlock-freedom with respect to \emph{binary} protocols, developed in prior work.
    Our decentralized, router-based analysis enables the sound and complete transference of protocol conformance and deadlock-freedom from \APCP to multiparty protocols.
\end{abstract}

{\small
\tableofcontents
}

\newpage

\section{Introduction}
\label{s:introduction}

This \ifthesis{}chapter \else{}paper \fi presents a new approach to the analysis of the \emph{protocols} that pervade concurrent and distributed software.
Such protocols provide an essential unifying glue between communicating programs; ensuring that communicating programs implement protocols correctly, avoiding protocol violations and deadlocks, is an important but difficult problem.
Here, we study \emph{multiparty session types (\MPST)}~\cite{journal/acm/HondaYC16}, an approach to correctness in message-passing programs that uses governing protocols as types in program verification.

As a motivating example, let us consider a \emph{recursive authorization protocol}, adapted from an example by Scalas and Yoshida~\cite{conf/popl/ScalasY19}.
It involves three participants: a Client, a Server, and an Authorization service.
Intuitively, the protocol proceeds as follows.
The Server requests the Client either to \emph{login} or to \emph{quit} the protocol.
In the case of login, the Client sends a password to the Authorization service, which then may authorize the login with the Server; subsequently, the protocol can be performed again: this is useful when, e.g., clients must renew their authorization privileges after some time.
In the case of quit, the protocol ends.

\MPST use \emph{global types} to specify multiparty protocols.
The authorization protocol just described can be specified by the following global type between Client ($c$), Server ($s$), and Authorization service~($a$):
\[
    \gtc{G_{\sff{auth}}} = \mu X \gtc. s {!} c \gtBraces*{
        \begin{array}{@{}l@{}}
            \sff{login} \gtc. c {!} a \gtBraces[\big]{ \msg \sff{passwd}<\sff{str}> \gtc. a {!} s \gtBraces{ \msg \sff{auth}<\sff{bool}> \gtc. X } } \gtc,
            \\
            \sff{quit} \gtc. c {!} a \gtBraces{ \sff{quit} \gtc. \tEnd }
        \end{array}
    }
    \label{eq:mpstAPCP:Gauth}
\]
After declaring a recursion  on the variable $X$ ($\mu X$), the global type $\gtc{G_{\sff{auth}}}$ stipulates that $s$ sends to $c$ ($s {!} c$) a label \sff{login} or \sff{quit}.
The rest of the protocol depends on this choice by $s$.
In the \sff{login}-branch, $c$ sends to $a$ a label \sff{passwd} along with a \sff{str}ing value ($\<\sff{str}\>$) and $a$ sends to $s$ a label \sff{auth} and a \sff{bool}ean value, after which the protocol loops to the beginning of the recursion ($X$).
In the \sff{quit}-branch, $c$ sends to $a$ a label \sff{quit} after which the protocol ends ($\tEnd$).

In \MPST, participants are implemented as distributed processes that communicate asynchronously.
Each process must correctly implement its corresponding portion of the protocol; these individual guarantees ensure that the interactions between processes conform to the given global type.
Correctness follows from \emph{protocol fidelity} (processes interact as stipulated by the protocol), \emph{communication safety} (no errors or mismatches in messages), and \emph{deadlock-freedom} (processes  never get stuck waiting for each other).
Ensuring that implementations satisfy these properties is a challenging problem, which is further compounded by two common and convenient features in interacting processes: \emph{delegation} and \emph{interleaving}.
We motivate them in the context of our example:
\begin{itemize}
    \item
        \emph{Delegation,} or higher-order channel passing, can effectively express that the Client transparently reconfigures its involvement by asking another participant (say, a Password Manager) to act on its behalf;

    \item
        \emph{Interleaving} arises when a single process implements more than one role: e.g., an implementation of both the Server and the Authorization service in a sequential process.
\end{itemize}
Note that while delegation is explicitly specified in a global type, interleaving arises in its implementation as interacting processes, not in its specification.

\MPST have been widely studied from foundational and applied angles~\cite{conf/concur/BettiniCDLDY08,journal/lmcs/CastagnaDP12,journal/lmcs/HuBYD12,conf/ecoop/ScalasDHY17,conf/ice/BarbaneraD19,journal/asep/BejleriDVEM19,conf/popl/ScalasY19,journal/tcs/CastellaniDGH20,conf/ecoop/ImaiNYY20,journal/pacmpl/MajumdarYZ20}.
The original theory by Honda \etal~\cite{conf/popl/HondaYC08} defines a behavioral type system~\cite{journal/csur/HuttelLVCCDMPRT16,journal/ftpl/AnconaBB0CDGGGH16} for a \picalc, which exploits linearity to ensure protocol fidelity and communication safety; most derived works retain this approach and target the same properties.
Deadlock-freedom is hard to ensure by typing when implementations feature delegation and interleaving.
In simple scenarios without interleaving and/or delegation, deadlock-freedom is easy, as it concerns a single-threaded protocol.
In contrast, deadlock-freedom for processes running multiple, interleaved protocols (possibly delegated) is a much harder problem, addressed only by some advanced type systems~\cite{conf/concur/BettiniCDLDY08,conf/coordination/PadovaniVV14,journal/mscs/CoppoDYP16}.

In this \ifthesis chapter\else paper\fi, we tackle the problem of ensuring that networks of interacting processes correctly implement a given global type in a Deadlock-free manner, while supporting delegation and  interleaving.
Our approach is informed by the differences between \emph{orchestration} and \emph{choreography}, two salient approaches to the coordination and organization of interacting processes in service-oriented paradigms~\cite{journal/computer/Peltz03}:
\begin{itemize}
    \item
        In \emph{orchestration}-based approaches, processes interact through a coordinator process which ensures that they all follow the protocol as intended.
        Quoting Van der Aalst, in an orchestration ``the conductor tells everybody in the orchestra what to do and makes sure they all play in sync''~\cite{book/vdAalst09}.

    \item
        In \emph{choreography}-based approaches, processes interact directly following the protocol without external coordination.
        Again quoting Van der Aalst, in a choreography ``dancers dance following a global scenario without a single point of control''~\cite{book/vdAalst09}.
\end{itemize}

Specification and analysis techniques based on \MPST fall under the choreography-based approach.
The global type provides the protocol's specification; based on the global type, participant implementations interact directly with each other, without an external coordinator.

As we will see, the contrast between orchestration and choreography is relevant here because it induces a different \emph{network topology} for interacting processes.
In an orchestration, the resulting process network is \emph{centralized}: all processes must connect to a central orchestrator process.
In a choreography, the process network is \emph{decentralized}, as processes can directly connect to each other.

\paragraph{Contributions.}

We develop a new decentralized analysis of multiparty  protocols.
\begin{itemize}
	\item
        Here `analysis' refers to (i)~ways of specifying such protocols as interacting processes \emph{and} (ii)~techniques to verify that those processes satisfy the intended correctness properties.

    \item
        Also, aligned with the above discussion, `decentralized' refers to the intended network topology for processes, which does not rely on an external coordinator.
\end{itemize}
Our decentralized analysis of global types enforces protocol fidelity, communication safety, and deadlock-freedom for process implementations, while uniformly supporting delegation, interleaving, and asynchronous communication.

\begin{figure}[t]
    \ifthesis\else\begin{mdframed}\fi
        \mbox{}\hfill%
        \begin{minipage}{.45\textwidth}
            \centering
            \ifthesis \begin{tikzpicture}
    \begin{scope}[local bounding box=bndP]
        \node (P) [minimum height=5.5mm, draw=black, shape=rectangle] at (0, 0) {$P$};
        \node (rtrP) [below=4.7mm of P.west, anchor=west] {router};
    \end{scope}
    \node [fit=(bndP), draw=black, shape=rectangle, inner sep=0mm] (rtdP) {};

    \begin{scope}[local bounding box=bndQ]
        \node (Q) [minimum height=5.5mm, right=4.7cm of P, draw=black, shape=rectangle] {$Q$};
        \node (rtrQ) [below=4.7mm of Q.east, anchor=east] {router};
    \end{scope}
    \node [fit=(bndQ), draw=black, shape=rectangle, inner sep=0mm] (rtdQ) {};

    \begin{scope}[local bounding box=bndR]
        \node (R) [minimum height=5.5mm, below right=3.4cm of P, draw=black, shape=rectangle] {$R$};
        \node (rtrR) [above=4.7mm of R.east, anchor=east] {router};
    \end{scope}
    \node [fit=(bndR), draw=black, shape=rectangle, inner sep=0mm] (rtdR) {};

    \draw [-] (bndP) to (bndQ);
    \draw [-] (bndQ) to (bndR);
    \draw [-] (bndR) to (bndP);
\end{tikzpicture}
            \else \begin{tikzpicture}
    \begin{scope}[local bounding box=bndP]
        \node (P) [minimum height=5.5mm, draw=black, shape=rectangle] at (0, 0) {$P$};
        \node (rtrP) [below=4.7mm of P.west, anchor=west] {router};
    \end{scope}
    \node [fit=(bndP), draw=black, shape=rectangle, inner sep=0mm] (rtdP) {};

    \begin{scope}[local bounding box=bndQ]
        \node (Q) [minimum height=5.5mm, right=4.7cm of P, draw=black, shape=rectangle] {$Q$};
        \node (rtrQ) [below=4.7mm of Q.east, anchor=east] {router};
    \end{scope}
    \node [fit=(bndQ), draw=black, shape=rectangle, inner sep=0mm] (rtdQ) {};

    \begin{scope}[local bounding box=bndR]
        \node (R) [minimum height=5.5mm, below right=3.4cm of P, draw=black, shape=rectangle] {$R$};
        \node (rtrR) [above=4.7mm of R.east, anchor=east] {router};
    \end{scope}
    \node [fit=(bndR), draw=black, shape=rectangle, inner sep=0mm] (rtdR) {};

    \draw [-] (bndP) to (bndQ);
    \draw [-] (bndQ) to (bndR);
    \draw [-] (bndR) to (bndP);
\end{tikzpicture} \fi
        \end{minipage}%
        \hfill\vline\hfill%
        \begin{minipage}{.45\textwidth}
            \centering
            \ifthesis \begin{tikzpicture}
    \node (O) [draw=black, shape=rectangle] at (0, 0) {
        $\begin{array}{c}
            \text{medium} \\ \text{or} \\ \text{arbiter}
        \end{array}$
    };
    \node (P) [above left=.5cm of O, minimum height=5.5mm, draw=black, shape=rectangle] {$P$};
    \node (Q) [above right=.5cm of O, minimum height=5.5mm, draw=black, shape=rectangle] {$Q$};
    \node (R) [below=.5cm of O, minimum height=5.5mm, draw=black, shape=rectangle] {$R$};

    \draw [-] (O) to (P);
    \draw [-] (O) to (Q);
    \draw [-] (O) to (R);
\end{tikzpicture}
            \else \begin{tikzpicture}
    \node (O) [draw=black, shape=rectangle] at (0, 0) {
        $\begin{array}{c}
            \text{medium} \\ \text{or} \\ \text{arbiter}
        \end{array}$
    };
    \node (P) [above left=.5cm of O, minimum height=5.5mm, draw=black, shape=rectangle] {$P$};
    \node (Q) [above right=.5cm of O, minimum height=5.5mm, draw=black, shape=rectangle] {$Q$};
    \node (R) [below=.5cm of O, minimum height=5.5mm, draw=black, shape=rectangle] {$R$};

    \draw [-] (O) to (P);
    \draw [-] (O) to (Q);
    \draw [-] (O) to (R);
\end{tikzpicture} \fi
        \end{minipage}%
        \hfill\mbox{}
    \ifthesis\else\end{mdframed}\fi

    \caption{Given processes $P$, $Q$, and $R$  implementing the roles of $c$, $s$, and $a$, respectively, \ifmpstAPCP \else of some \fi protocol~$\gtc{G_{\sff{auth}}}$\ifmpstAPCP\else, the protocol \fi can be realized as a choreography of routed implementations (our approach, left) and as an orchestration of implementations, with a medium or arbiter process (previous works, right).}
    \ifmpstAPCP \label{f:mpstAPCP:chorOrch}
    \else \label{f:intro:mpstAPCP:diag} \fi
\end{figure}

The \emph{key idea} of our analysis is to exploit global types to generate \emph{router processes} (simply \emph{routers}) that enable  participant implementations to communicate directly.
There is a router per participant; it is intended to serve as a ``wrapper'' for an individual participant's implementation.
The composition of an implementation with its corresponding router is called a \emph{routed implementation}.
Collections of routed implementations can then be connected in arbitrary \emph{process networks} that correctly realize the multiparty protocol, subsuming centralized and decentralized topologies.

Routers are \emph{synthesized} from global types, and do not change the behavior of the participant implementations they wrap; they merely ensure that networks of routed implementations correctly behave as described by the given multiparty protocol.
Returning to Van der Aalst's analogies quoted above, we may say that in our setting participant implementations are analogous to skilled but barefoot dancers, and that routers provide them with the appropriate shoes to dance without a central point of control.
To make this analogy a bit more concrete, \Cref{f:mpstAPCP:chorOrch} (left) illustrates the decentralized process network formed by routed implementations of the participants of $\gtc{G_{\sff{auth}}}$: once wrapped by an appropriate router, implementations $P$, $Q$, and $R$ can be composed directly in a decentralized process network.

A central technical challenge in our approach is to ensure that compositions of routed implementations conform to their global type.
The channels that enable the arbitrary composition of routed implementations need to be typed in accordance with the given multiparty protocol.
Unfortunately, the usual notion of projection in \MPST, which obtains a single participant's perspective from a global type, does not suffice: we need a local perspective that is relative to the \emph{two participants} that the connected routed implementations represent.
To this end, we introduce a new notion, \emph{relative projection}, which isolates the exchanges of the global type that relate to pairs of participants.
In the case of $\gtc{G_{\sff{auth}}}$, for instance, we need three relative types, describing the protocol for $a$ and $c$, for $a$ and $s$, and for $c$ and $s$.

A derived challenge is that when projecting a global type onto a pair of participants, it is possible to encounter \emph{non-local choices}: choices by other participants that affect the protocol between the two participants involved in the projection.
To handle this, relative projection explicitly records non-local choices in the form of \emph{dependencies}, which inform the projection's participants that they need to coordinate on the results of the non-local choices.

In all, our decentralized analysis of global types relies on three intertwined novel notions:
\begin{itemize}
    \item
        \textbf{Routers} that wrap participant implementations in order to enable their composition in arbitrary network topologies, whilst guaranteeing that the resulting process networks correctly follow the given global type in a Deadlock-free manner.

    \item
        \textbf{Relative Types} that type the channels between routed implementations, obtained by means of a new notion of projection of global types onto pairs of participants.

    \item
        \textbf{Relative Projection} and \textbf{Dependencies} that make it explicit in relative types that participants need to coordinate on non-local choices.
\end{itemize}
The key ingredients of our decentralized analysis for $\gtc{G_{\sff{auth}}}$ are jointly depicted in \Cref{f:mpstAPCP:workflow}.

With respect to prior analyses of multiparty protocols, a distinguishing feature of our work is its natural support of decentralized process networks, as expected in a choreography-based approach.
Caires and Pérez~\cite{conf/forte/CairesP16} connect participant implementations through a central coordinator, called \emph{medium process}.
This medium process is generated from a global type, and intervenes in all exchanges to ensure that the participant implementations follow the multiparty protocol.
The composition of the medium with the participant implementations can then be analyzed using a type system for binary sessions.
In a similar vein, Carbone \etal~\cite{conf/concur/CarboneLMSW16} define a type system in which they use global types to validate choreographies of participant implementations.
Their analysis of protocol implementations---in particular, deadlock-freedom---relies on encodings into another type system where participant implementations connect to a central coordinator, called the \emph{arbiter process}.
Similar to mediums, arbiters are generated from the global type to ensure that participant implementations follow the protocol as intended.
Both these approaches are clear examples of orchestration, and thus do not support decentralized network topologies.

To highlight the differences between our decentralized analysis and prior approaches, compare the choreography of routed implementations in \Cref{f:mpstAPCP:chorOrch}~(left) with an implementation of $\gtc{G_{\sff{auth}}}$ in the style of Caires and Pérez and of Carbone \etal, in \Cref{f:mpstAPCP:chorOrch}~(right).
These prior works rely on orchestration because the type systems they use for verifying process implementations restrict connections between processes: they only admit a form of process composition that makes it impossible to simultaneously connect three or more participant implementations~\cite{journal/jlamp/DardhaP22,conf/express/DardhaP15}.
In this \ifthesis{}chapter\else{}paper\fi, we overcome this obstacle by relying on \APCP (Asynchronous Priority-based Classical Processes)~\cite{report/vdHeuvelP21B}, a type system that allows for more general forms of process composition.
By using annotations on types, \APCP prevents \emph{circular dependencies}, i.e., cyclically connected processes that are stuck waiting for each other.
This is how our approach supports networks of routed participants in both centralized and decentralized topologies, thus subsuming  choreography and orchestration approaches.

\begin{figure}[t]
    \ifthesis\else\begin{mdframed}\fi
        \begin{center}
            \ifthesis 
\begin{tikzpicture}
    \begin{scope}[local bounding box=BNDc]
        \node [align=center] (G) at (0,0) {$\gtc{G_{\sff{auth}}}$};

        \node [below=15mm of G,xshift=-5mm,anchor=center,align=center] (L) {$L_c$};
        \draw [-to] (G) -- (L) node [midway,left,text=cblRed,align=center] {local projection \\ (\secref{s:mpstAPCP:analysis:typing:locProj})};

        \node [below=15mm of G,xshift=5mm,anchor=center,align=center] (R) {$\rtc{R_{cs}},\rtc{R_{ca}}$};
        \draw [-to] (G) -- (R) node [midway,right,text=cblRed,align=center] {relative projection \\ (\defref{d:mpstAPCP:relProj})};

        \node [below=15mm of L,anchor=center,align=center] (P) {$P$};
        \draw [|-|] (L) -- (P) node [midway,left,text=cblRed,align=center] {type check \\ in APCP (\ifappendix\secref{s:mpstAPCP:APCP}\else{}Ch.\ \labelcref{c:APCP}\fi)};

        \node [below=15mm of R,anchor=center,align=center] (RTR) {$\rtr{c}$};
        \draw [-to] (R) -- (RTR) node [midway,right,text=cblRed,align=center] {router synthesis \\ (\secref{s:mpstAPCP:analysis:routers})};
    \end{scope}

    \node [right=60mm of G,anchor=center,align=center] (G2) {$\gtc{G_{\sff{auth}}}$};

    \begin{scope}[local bounding box=BNDnet]
        \node [below=20mm of G2,xshift=-18mm,anchor=center,align=center,draw=black] (RIMPLc) {$P \| \rtr{c}$};
        \node [below=3mm of RIMPLc.center,anchor=north,align=center] {client};
        \node [above=2mm of RIMPLc.center,anchor=south,align=center] {\vphantom{l}};
        \node [left=0mm of RIMPLc.west,anchor=east] {\hphantom{s}};

        \node [below=20mm of G2,anchor=center,align=center,draw=black] (RIMPLs) {$Q \| \rtr{s}$};
        \node [below=3mm of RIMPLs.center,anchor=north,align=center] {server\vphantom{l}};
        \node [above=2mm of RIMPLs.center,anchor=south,align=center] {\vphantom{l}};

        \node [below=20mm of G2,xshift=18mm,anchor=center,align=center,draw=black] (RIMPLa) {$R \| \rtr{a}$};
        \node [below=3mm of RIMPLa.center,anchor=north,align=center] {authorization \\ service};
        \node [above=2mm of RIMPLa.center,anchor=south,align=center] {\vphantom{l}};
        \node [right=0mm of RIMPLc.east,anchor=west] {\hphantom{s}};
        \node [below=10mm of RIMPLa.center,anchor=north,align=center] {\vphantom{l}};
    \end{scope}
    \node [fit=(BNDnet),draw=black,dashed,inner xsep=1mm,inner ysep=0mm,rounded corners=2mm] (BOXnet) {};
    \node [below=1mm of BOXnet.south,anchor=north,align=center] {network of routed implementations \\ of $\gtc{G_{\sff{auth}}}$ (\defref{d:mpstAPCP:networks})};

    \draw [-to] (G2) -- ([yshift=1mm,xshift=3mm]RIMPLc.north);
    \draw [-to] (G2) -- ([yshift=1mm,xshift=0mm]RIMPLs.north);
    \draw [-to] (G2) -- ([yshift=1mm,xshift=-3mm]RIMPLa.north);

    \node [fit=(BNDc),draw=black,inner xsep=3mm,inner ysep=0mm,rectangle callout,rounded corners=2mm,callout absolute pointer={(RIMPLc.west)}] (BOXc) {};
    \node [below=1mm of BOXc.south,anchor=north,align=center] {routed implementation of $c$ (\defref{d:mpstAPCP:routedImplementations})};
\end{tikzpicture}
            \else 
\begin{tikzpicture}
    \begin{scope}[local bounding box=BNDc]
        \node [align=center] (G) at (0,0) {$\gtc{G_{\sff{auth}}}$};

        \node [below=15mm of G,xshift=-5mm,anchor=center,align=center] (L) {$L_c$};
        \draw [-to] (G) -- (L) node [midway,left,text=cblRed,align=center] {local projection \\ (\secref{s:mpstAPCP:analysis:typing:locProj})};

        \node [below=15mm of G,xshift=5mm,anchor=center,align=center] (R) {$\rtc{R_{cs}},\rtc{R_{ca}}$};
        \draw [-to] (G) -- (R) node [midway,right,text=cblRed,align=center] {relative projection \\ (\defref{d:mpstAPCP:relProj})};

        \node [below=15mm of L,anchor=center,align=center] (P) {$P$};
        \draw [|-|] (L) -- (P) node [midway,left,text=cblRed,align=center] {type check \\ in APCP (\ifappendix\secref{s:mpstAPCP:APCP}\else{}Ch.\ \labelcref{c:APCP}\fi)};

        \node [below=15mm of R,anchor=center,align=center] (RTR) {$\rtr{c}$};
        \draw [-to] (R) -- (RTR) node [midway,right,text=cblRed,align=center] {router synthesis \\ (\secref{s:mpstAPCP:analysis:routers})};
    \end{scope}

    \node [right=60mm of G,anchor=center,align=center] (G2) {$\gtc{G_{\sff{auth}}}$};

    \begin{scope}[local bounding box=BNDnet]
        \node [below=20mm of G2,xshift=-18mm,anchor=center,align=center,draw=black] (RIMPLc) {$P \| \rtr{c}$};
        \node [below=3mm of RIMPLc.center,anchor=north,align=center] {client};
        \node [above=2mm of RIMPLc.center,anchor=south,align=center] {\vphantom{l}};
        \node [left=0mm of RIMPLc.west,anchor=east] {\hphantom{s}};

        \node [below=20mm of G2,anchor=center,align=center,draw=black] (RIMPLs) {$Q \| \rtr{s}$};
        \node [below=3mm of RIMPLs.center,anchor=north,align=center] {server\vphantom{l}};
        \node [above=2mm of RIMPLs.center,anchor=south,align=center] {\vphantom{l}};

        \node [below=20mm of G2,xshift=18mm,anchor=center,align=center,draw=black] (RIMPLa) {$R \| \rtr{a}$};
        \node [below=3mm of RIMPLa.center,anchor=north,align=center] {authorization \\ service};
        \node [above=2mm of RIMPLa.center,anchor=south,align=center] {\vphantom{l}};
        \node [right=0mm of RIMPLc.east,anchor=west] {\hphantom{s}};
        \node [below=10mm of RIMPLa.center,anchor=north,align=center] {\vphantom{l}};
    \end{scope}
    \node [fit=(BNDnet),draw=black,dashed,inner xsep=1mm,inner ysep=0mm,rounded corners=2mm] (BOXnet) {};
    \node [below=1mm of BOXnet.south,anchor=north,align=center] {network of routed implementations \\ of $\gtc{G_{\sff{auth}}}$ (\defref{d:mpstAPCP:networks})};

    \draw [-to] (G2) -- ([yshift=1mm,xshift=3mm]RIMPLc.north);
    \draw [-to] (G2) -- ([yshift=1mm,xshift=0mm]RIMPLs.north);
    \draw [-to] (G2) -- ([yshift=1mm,xshift=-3mm]RIMPLa.north);

    \node [fit=(BNDc),draw=black,inner xsep=3mm,inner ysep=0mm,rectangle callout,rounded corners=2mm,callout absolute pointer={(RIMPLc.west)}] (BOXc) {};
    \node [below=1mm of BOXc.south,anchor=north,align=center] {routed implementation of $c$ (\defref{d:mpstAPCP:routedImplementations})};
\end{tikzpicture} \fi
        \end{center}
    \ifthesis\else\end{mdframed}\fi
    \caption{
        Decentralized analysis of $\gtc{G_{\sff{auth}}}$ into a network of routed implementations.
        The definition of $\gtc{G_{\sff{auth}}}$ contains message types.
        Focusing on the client $c$ (on the left), $L_c$ denotes a session type, whereas $\rtc{R_{cs}}$ and $\rtc{R_{ca}}$ are relative types with respect to the server and the authorization service, respectively.
        \label{f:mpstAPCP:workflow}
    }
\end{figure}

\newpage
\paragraph{Outline.}

This \ifthesis{}chapter \else{}paper \fi is structured as follows:
\begin{itemize}
    \item
        \ifappendix
        \Cref{s:mpstAPCP:APCP} recalls \APCP and summarizes the correctness properties for asynchronous processes derived from typing.
        \else
        \Cref{s:mpstAPCP:APCPExample} gives an example of a decentralized implementation of our running authorization protocol in \APCP.
        \fi
    \item
        \ifthesis \Cref{s:mpstAPCP:mpst} recalls \else \Cref{s:mpst} introduces \fi relative types and relative projection, and defines \emph{well-formed} global types, a class of global types that includes protocols with non-local choices.
    \item
        \Cref{s:mpstAPCP:analysis} introduces the synthesis of routers.
        A main result is their typability in \APCP (\Cref{t:mpstAPCP:routerTypes}).
        We establish deadlock-freedom for networks of routed implementations (\Cref{t:mpstAPCP:globalDf}), which we transfer to multiparty protocols via an operational correspondence result (\Cref{t:mpstAPCP:completeness,t:mpstAPCP:soundness}).
        Moreover, we show that our approach strictly generalizes prior analyses based on centralized topologies (\Cref{t:mpstAPCP:mediumBisim}).
    \item
        \Cref{s:mpstAPCP:examples} demonstrates our contributions with a full development of the routed implementations for $\gtc{G_{\sff{auth}}}$, and an example of the flexible support for \emph{delegation} and \emph{interleaving} enabled by our router-based approach and \APCP.
\end{itemize}
We discuss further related works in \Cref{s:mpstAPCP:rw} and conclude the \ifthesis{}chapter \else{}paper \fi in \Cref{s:mpstAPCP:concl}.
We use colors to improve readability.

\section{\APCP: Asynchronous Processes, Deadlock-free by Typing}
\label{s:apcp}

We recall \APCP as defined by Van den Heuvel and P\'{e}rez~\cite{report/vdHeuvelP21B}.
\APCP is a type system for asynchronous $\pi$-calculus processes (with non-blocking outputs)~\cite{conf/ecoop/HondaT91,report/Boudol92}, with support for recursion and cyclic connections.
In this type system, channel endpoints are assigned linear types that represent two-party (binary) \emph{session types}~\cite{conf/concur/Honda93}.
Well-typed \APCP processes  {preserve typing} (\Cref{t:subjectRed}) and are {deadlock-free} (\Cref{t:dlFree}).

At its basis, \APCP combines Dardha and Gay's Priority-based Classical Processes (\PCP)~\cite{conf/fossacs/DardhaG18} with DeYoung \etal's continuation-passing semantics for asynchrony~\cite{conf/csl/DeYoungCPT12}, and adds recursion,  inspired by the work of Toninho \etal~\cite{conf/tgc/ToninhoCP14}.
We refer the interested reader to the work by Van den Heuvel and P\'{e}rez~\cite{report/vdHeuvelP21B} for a motivation of design choices and proofs of results.

\paragraph{Process Syntax}

\begin{figure}[t]
    \begin{mdframed}
        Process syntax:
        \begin{align*}
            P,Q \deq
            &~ x[y,z]
            & \text{output}
            & ~~~ \sepr ~~
              x(y,z) \sdot P
            & \text{input}
            \\[-3pt]
            \sepr\!
            &~ x[z] \triangleleft i
            & \text{selection}
            & ~~~ \sepr ~~
              x(z) \triangleright \{i: P_i\}_{i \in I}
            & \text{branching}
            & ~~~ \sepr ~~
              \nu{x y}P
            & \text{restriction}
            \\[-3pt]
            \sepr\!
            &~ (P \| Q)
            & \text{parallel}
            & ~~~ \sepr ~~
              \0
            & \text{inaction}
            & ~~~ \sepr ~~
              x \fwd y
            & \text{forwarder}
            \\[-3pt]
            \sepr\!
            &~ \mu X(\tilde{z}) \sdot P
            & \text{recursive loop}
            & ~~~ \sepr ~~
              X\call{\tilde{z}}
            & \text{recursive call}
        \end{align*}

        \vspace{-1em}
        \hbox to \textwidth{\leaders\hbox to 3pt{\hss . \hss}\hfil}

        \smallskip
        Structural congruence:
        \begin{align*}
            P \equiv_\alpha P' \implies {}
            &
            &
            P
            &\equiv P'
            &
            x \fwd y
            &\equiv y \fwd x
            \\
            &
            &
            P \| Q
            &\equiv Q \| P
            &
            \nu{x y} x \fwd y
            &\equiv \0
            \\
            &
            &
            P \| \0
            &\equiv P
            &
            P \| (Q \| R)
            &\equiv (P \| Q) \| R
            \\
            x,y \notin \fn(P) \implies {}
            &
            &
            P \| \nu{x y}Q
            &\equiv \nu{x y}(P \| Q)
            &
            \nu{x y}\0
            &\equiv \0
            \\
            |\tilde{z}| = |\tilde{y}| \implies {}
            &
            &
            \mu X(\tilde{z}) \sdot P
            &\equiv P \big\{(\mu X(\tilde{y}) \sdot P \subst{\tilde{y}/\tilde{z}}) / X\call{\tilde{y}}\big\}
            &
            \nu{x y}P
            &\equiv \nu{y x}P
            \\
            &
            &
            &
            &
            \nu{x y}\nu{z w} P
            &\equiv \nu{z w}\nu{x y} P
        \end{align*}

        \vspace{-1em}
        \hbox to \textwidth{\leaders\hbox to 3pt{\hss . \hss}\hfil}

        \smallskip
        Reduction:
        \begin{align*}
            & \rred{\scc{Id}}
            &
            z,y \neq x \implies {}
            &
            &
            \nu{y z}( x \fwd y \| P )
            &\redd P \subst{x/z}
            \\
            & \rred{\tensor\parr}
            &
            &
            &
            \nu{x y}( x[a,b] \| y(v,z) \sdot P )
            & \redd P \subst{a/v, b/z}
            \\
            & \rred{\oplus\&}
            &
            j \in I \implies {}
            &
            &
            \nu{x y}( x[b] \triangleleft j \| y(z) \triangleright \{i: P_i\}_{i \in I} )
            &\redd P_j \subst{b/z}
        \end{align*}
        \begin{mathpar}
            \inferrule*[right=\rLab{$\rred{\equiv}$},vcenter]
            {
                P \equiv P'
                \\
                P' \redd Q'
                \\
                Q' \equiv Q
            }
            { P \redd Q }
            \and
            \inferrule*[right=\rLab{$\rred{\onu}$},vcenter]
            { P \redd Q }
            { \nu{x y} P \redd \nu{x y} Q }
            \and
            \inferrule*[right=\rLab{$\rred{\|}$},vcenter]
            { P \redd Q }
            { P \| R \redd Q \| R }
        \end{mathpar}
    \end{mdframed}

    \caption{Definition of \APCP's process language.}
    \label{f:procdef}
\end{figure}

We write $x, y, z, \ldots$ to denote (channel) \emph{endpoints} (also known as \emph{names}), and write $\tilde{x}, \tilde{y}, \tilde{z}, \ldots$ to denote sequences of endpoints.
With a slight abuse of notation, we sometimes write $x_i \in  \tilde{x}$ to refer to a specific element in the sequence $\tilde{x}$.
Also, we write $i, j, k, \ldots$ to denote \emph{labels} for choices and $I, J, K, \ldots$ to denote sets of labels.
We write $X, Y, \ldots$ to denote \emph{recursion variables}, and $P,Q, \ldots$ to denote processes.

\Cref{f:procdef} (top) gives the syntax of processes.
The output action $x[y,z]$ sends a message $y$ (an endpoint) and a continuation endpoint $z$ along $x$.
The input prefix $x(y,z) \sdot  P$ blocks until a message and a continuation endpoint are received on $x$ (referred to in $P$ as the placeholders $y$ and $z$, respectively), binding $y$ and $z$ in $P$.
The selection action $x[z] \puts i$ sends a label $i$ and a continuation endpoint $z$ along $x$.
The branching prefix $x(z) \gets \{i: P_i\}_{i \in I}$ blocks until it receives a label $i \in I$ and a continuation endpoint (reffered to in $P_i$ as the placeholder $z$) on $x$, binding $z$ in each $P_i$.
Restriction $\nu{x y} P$ binds $x$ and $y$ in $P$, thus declaring them as the two endpoints of the same channel and enabling communication, as in Vasconcelos~\cite{journal/ic/Vasconcelos12}.
The process $P \| Q$ denotes the parallel composition of $P$ and $Q$.
The process $\0$ denotes inaction.
The forwarder process $x \fwd y$ is a primitive copycat process that links together $x$ and $y$.
The prefix $\mu X(\tilde{z}) \sdot P$ defines a recursive loop, binding occurrences of $X$ in $P$; the endpoints $\tilde{z}$ form a context for $P$.
The recursive call $X\call{\tilde{z}}$ loops to its corresponding $\mu X$, providing the endpoints $\tilde{z}$ as context.
We only consider contractive recursion, disallowing processes with subexpressions of the form $\mu X_1(\tilde{z}) \ldots \mu X_n(\tilde{z}) \sdot X_1\call{\tilde{z}}$.

Endpoints and recursion variables are free unless otherwise stated (i.e., unless they are bound somehow).
We write $\fn(P)$ and $\frv(P)$ for the sets of free names and free recursion variables of $P$, respectively.
Also, we write $P \subst{x/y}$ to denote the capture-avoiding substitution of the free occurrences of $y$ in $P$ for $x$.
The notation $P\rsubst{\mu X(\tilde{y}) \sdot P' / X\call{\tilde{y}}}$ denotes the substitution of occurrences of recursive calls $X\call{\tilde{y}}$ in $P$ with the recursive loop $\mu X(\tilde{y}) \sdot P'$, which we call \emph{unfolding} recursion.
We write sequences of substitutions $P \subst{x_1/y_1} \ldots \subst{x_n/y_n}$ as $P \subst{x_1/y_1, \ldots, x_n/y_n}$.

In an output $x[y,z]$, both $y$ and $z$ are free; they can be bound to a continuation process using parallel composition and restriction, as in $\nu{y a}\nu{z b}(x[y,z] \| P_{a,b})$.
The same applies to selection $x[z] \puts i$.
We introduce useful notations that elide the restrictions and continuation endpoints:

\begin{notation}[Derivable Actions and Prefixes]\label{n:sugar}
    We use the following syntactic sugar:
    \begin{align*}
        \ol{x}[y] \cdot P
        & \deq \nu{y a}\nu{z b}(x[a,b] \| P\subst{z/x})
        &
        \ol{x} \puts \ell \cdot P
        & \deq \nu{z b}(x[b] \puts \ell \| P\subst{z/x})
        \\
        x(y) \sdot P
        & \deq x(y,z) \sdot P\subst{z/x}
        &
        x \gets \{i:P_i\}_{i \in I}
        & \deq x(z) \gets \{i:P_i\subst{z/x}\}_{i \in I}
    \end{align*}
\end{notation}

\noindent
Note the use of `$\,\cdot\,$' instead of `$\,\sdot\,$' in output and selection to stress that they are non-blocking.

\paragraph{Operational Semantics}

We define a reduction relation for processes ($P \redd Q$) that formalizes how complementary actions on connected endpoints may synchronize.
As usual for $\pi$-calculi, reduction relies on \emph{structural congruence} ($P \equiv Q$), which relates processes with minor syntactic differences; it is the smallest congruence on the syntax of processes (\figref{f:procdef} (top)) satisfying the axioms in \Cref{f:procdef} (middle).

Structural congruence defines the following properties of our process language.
Processes are equivalent up to $\alpha$-equivalence.
Parallel composition is associative and commutative, with unit $\0$.
The forwarder process is symmetric, and equivalent to inaction if both endpoints are bound together through restriction.
A parallel process may be moved into or out of a restriction as long as the bound channels do not occur free in the moved process: this is \emph{scope inclusion} and \emph{scope extrusion}, respectively.
Restrictions on inactive processes may be dropped, and the order of endpoints in restrictions and of consecutive restrictions does not matter.
Finally, a recursive loop is equivalent to its unfolding, replacing any recursive calls with copies of the recursive loop, where the call's endpoints are pairwise substituted for the contextual endpoints of the loop.

As we will see next, the semantics of \APCP is closed under structural congruence.
This means that processes are \emph{equi-recursive}; however, \APCP's typing discipline (described next) treats recursive types as \emph{iso-recursive} (see, e.g., Pierce~\cite{book/Pierce02}).

We can now define our reduction relation.
We define the reduction relation $P \redd Q$ by the axioms and closure rules in \Cref{f:procdef} (bottom).
Presentations of Curry-Howard interpretations of linear logic often include \emph{commuting conversions} (such as~\cite{conf/concur/CairesP10,conf/icfp/Wadler12,conf/fossacs/DardhaG18,conf/csl/DeYoungCPT12}), which allow pulling prefixes on free channels out of restrictions; they are not necessary for deadlock-freedom in \APCP, so we do not include them.

Rule~$\rred{\scc{Id}}$ implements the forwarder as a substitution.
Rule~$\rred{\tensor \parr}$ synchronizes an output and an input on connected endpoints and substitutes the message and continuation endpoint.
Rule~$\rred{\oplus \&}$ synchronizes a selection and a branch:
the received label determines the continuation process, substituting the continuation endpoint appropriately.
Rules~$\rred{\equiv}$, $\rred{\onu}$, and $\rred{\|}$ close reduction under structural congruence, restriction, and parallel composition, respectively.

\begin{notation}[Reductions]\label{n:redd}
    We write  $\redd^\ast$ for the reflexive, transitive closure of $\redd$.
    Also, we write $P \redd^\star Q$ if $P \redd^\ast Q$ in a finite number of steps, and $P \nredd^\ast Q$ for the non-existence of a series of reductions from $P$ to $Q$.
\end{notation}

\paragraph{Session Types}

\begin{table}[t]
    \begin{mdframed}
        \begin{tabularx}{\textwidth}{l||X}
            \textbf{Session Type} & \textbf{Endpoint Behavior} \\ \hline
            \\[-7pt]
            $A \tensor^\pri B$ & output an endpoint of type $A$, then behave as $B$ \\
            $A \parr^\pri B$ & input an endpoint of type $A$, then behave as $B$ \\
            $\oplus^\pri \{ i: A_i \}_{i \in I}$ & select a label $i \in I$, then behave as $A_i$ \\
            $\&^\pri \{ i: A_i \}_{i \in I}$ & receive a choice for a label $i \in I$, then behave as $A_i$ \\
            $\bullet$ & closed session; no behavior
        \end{tabularx}
    \end{mdframed}

    \caption{Session types and their associated endpoint behaviors (cf.\ \Cref{d:props}).}
    \label{tbl:types}
\end{table}

\APCP types processes by assigning binary session types to channel endpoints.
Following Curry-Howard interpretations, we present session types as linear logic propositions (cf., e.g., Caires \etal~\cite{journal/mscs/CairesPT16}, Wadler~\cite{conf/icfp/Wadler12}, Caires and P\'{e}rez~\cite{conf/esop/CairesP17}, and Dardha and Gay~\cite{conf/fossacs/DardhaG18}).
We extend these propositions with recursion and \emph{priority} annotations on connectives.
Intuitively, actions typed with lower priority should be performed before those with higher priority.

We write $\pri, \opri, \pi, \rho, \ldots$ to denote priorities, and $\omega$ to denote the ultimate priority that is greater than all other priorities  and cannot be increased further.
That is, $\forall t \in \mbb{N}.~\omega > t$ and $\forall t \in \mbb{N}.~\omega + t = \omega$.

\begin{definition}[Session Types]\label{d:props}
    The following grammar defines the syntax of \emph{session types} $A,B$.
    Let $\pri \in \mbb{N} \cup \{\omega\}$.
    \begin{align*}
        A,B
        &\deq A \tensor^\pri B \sepr A \parr^\pri B \sepr \oplus^\pri \{i: A\}_{i \in I} \sepr \&^\pri \{i: A\}_{i \in I} \sepr \bullet \sepr \mu X \sdot A \sepr X
    \end{align*}
\end{definition}

\noindent
\Cref{tbl:types} gives {session types} and the behavior that is expected of an endpoint with each type  (recursive types entail a communication behavior only after unfolding).
Note that $\bullet$ does not require a priority, as closed endpoints do not exhibit behavior and thus are non-blocking.
We define $\bullet$ as a single, self-dual type for closed endpoints (cf.\ Caires~\cite{report/Caires14} and Atkey \etal~\cite{chapter/AtkeyLM16}).

Type $\mu X \sdot A$ denotes a recursive type, in which $A$ may contain occurrences of the recursion variable $X$.
As customary, $\mu$ is a binder: it induces the standard notions of $\alpha$-equivalence, substitution (denoted $A\subst{B/X}$), and free recursion variables (denoted $\frv(A)$).
We work with tail-recursive, contractive types, disallowing types of the form $\mu X_1 \ldots \mu X_n \sdot X_1$ and $\mu X \sdot X \tensor^\pri A$.
Recursive types are treated iso-recursively: there will be an explicit typing rule that unfold recursive types, and recursive types are not equal to their unfolding.
We postpone formalizing the unfolding of recursive types, as it requires additional definitions to ensure consistency of priorities upon unfolding.

\emph{Duality}, the cornerstone of session types and linear logic, ensures that the two endpoints of a channel have matching actions.
Furthermore, dual types must have matching priority annotations.
The following inductive definition of duality suffices for our tail-recursive types (cf.\ Gay \etal~\cite{conf/places/GayTV20}).

\begin{definition}[Duality]\label{d:duality}
    The \emph{dual} of session type $A$, denoted $\ol{A}$, is defined inductively as follows:
    \begin{align*}
        \ol{A \tensor^\pri B}
        &\deq \ol{A} \parr^\pri \ol{B}
        &
        \ol{\oplus^\pri \{ i: A_i \}_{i \in I}}
        &\deq \&^\pri \{ i: \ol{A_i} \}_{i \in I}
        &
        \ol{\bullet}
        &\deq \bullet
        &
        \ol{\mu X \sdot A}
        &\deq \mu X \sdot \ol{A}
        \\
        \ol{A \parr^\pri B}
        &\deq \ol{A} \tensor^\pri \ol{B}
        &
        \ol{\&^\pri \{ i: A_i \}_{i \in I}}
        &\deq \oplus^\pri \{ i: \ol{A_i} \}_{i \in I}
        &
        &
        &
        \ol{X}
        &\deq X
    \end{align*}
\end{definition}

The priority of a type is determined by the priority of the type's outermost connective:

\begin{definition}[Priorities]\label{d:priority}
    For session type $A$, $\pr(A)$ denotes its \emph{priority}:
    \begin{align*}
        \pr(A \tensor^\pri B)
        \deq \pr(A \parr^\pri B)
        &\deq \pri
        &
        \pr(\mu X \sdot A)
        &\deq \pr(A)
        \\
        \pr(\oplus^\pri\{i:A_i\}_{i \in I})
        \deq \pr(\&^\pri\{i:A_i\}_{i \in I})
        &\deq \pri
        &
        \pr(\bullet)
        \deq \pr(X)
        &\deq \omega
    \end{align*}
\end{definition}

\noindent
The priority of $\bullet$ and $X$ is $\omega$: they denote ``final'', non-blocking actions of protocols.
Although $\tensor$ and $\oplus$ also denote non-blocking actions, their priority is not constant:
duality ensures that the priority for $\tensor$ (resp.\ $\oplus$) matches the priority of a corresponding $\parr$ (resp.\ $\&$), which  denotes a blocking action.

We now turn to formalizing the unfolding of recursive types.
Recall the intuition that actions typed with lower priority should be performed before those with higher priority.
Based on this rationale, we observe that the unfolding of the recursive type $\mu X \sdot A$ should not result in  $A\subst{\mu X \sdot A/X}$, as usual, but that the priorities of the unfolded type should be \emph{increased}.
This is because the actions related to the unfolded recursion should be performed \emph{after} the prefix.
For example, consider the recursive type $\mu X \sdot A \parr^0 X$.
If we unfold this type without increasing the priority, we would obtain $A \parr^0 (\mu X \sdot A \parr^0 X)$, a type in which the priorities no longer determine a global ordering between the two inputs.
By increasing the priority in the unfolded type as in, e.g., $A \parr^0 (\mu X \sdot A \parr^1 X)$, a global ordering is preserved.

We make this intuition precise by defining the \emph{lift} of priorities in types:

\begin{definition}[Lift]\label{d:lift}
    For proposition $A$ and $t \in \mbb{N}$, we define $\lift{t}A$ as the \emph{lift} operation:
    \begin{align*}
        \lift{t}(A \tensor^\pri B)
        &\deq (\lift{t}A) \tensor^{\pri+t} (\lift{t}B)
        &
        \lift{t}(\oplus^\pri \{i: A_i\}_{i \in I})
        &\deq \oplus^{\pri+t} \{i: \lift{t}A_i\}_{i \in I}
        &
        \lift{t}\bullet
        &\deq \bullet
        \\
        \lift{t}(A \parr^\pri B)
        &\deq (\lift{t}A) \parr^{\pri+t} (\lift{t}B)
        &
        \lift{t}(\&^\pri \{i: A_i\}_{i \in I})
        &\deq \&^{\pri+t} \{i: \lift{t}A_i\}_{i \in I}
        \\
        \lift{t}(\mu X \sdot A)
        &\deq \mu X \sdot (\lift{t}A)
        &
        \lift{t}X
        &\deq X
    \end{align*}
\end{definition}

\noindent
Henceforth, the unfolding of  $\mu X \sdot A$ is $A\subst{\mu X \sdot (\lift{t} A)/X}$, denoted $\unfold^t(\mu X \sdot A)$, where $t \in \mbb{N}$ depends on the highest priority of the types occurring in a typing context.
We recall that we do not consider types to be equi-recursive: recursive types are not equal to their unfolding.
Recursive types can only be unfolded by typing rules, discussed next.

We now define the highest priority of a type:

\begin{definition}[Highest Priority]\label{d:maxpr}
    For session type $A$, $\maxpr(A)$ denotes its \emph{highest priority}:
    \begin{align*}
        \maxpr(A \tensor^\pri B) \deq \maxpr(A \parr^\pri B)
        &\deq \max(\maxpr(A), \maxpr(B), \pri)
        \\
        \maxpr(\oplus^\pri \{i: A_i\}_{i \in I}) \deq \maxpr(\&^\pri \{i: A_i\}_{i \in I})
        &\deq \max(\max_{i \in I}(\maxpr(A_i)), \pri)
        \\
        \maxpr(\mu X \sdot A)
        &\deq \maxpr(A)
        \\
        \maxpr(\bullet) \deq \maxpr(X)
        &\deq 0
    \end{align*}
\end{definition}

\noindent
Notice how the highest priority of $\bullet$ and $X$ is 0, in contrast to their priority (as given by \Cref{d:priority}): they do not contribute to the increase in priority needed for unfolding recursive types.

\paragraph{Typing Rules}

The typing rules of \APCP ensure that actions with lower priority are performed before those with higher priority (cf.\ Dardha and Gay~\cite{conf/fossacs/DardhaG18}).
To this end, they enforce the following laws:
\begin{enumerate}
    \item
        an action with priority $\pri$ must be prefixed only by inputs and branches with priority strictly smaller than $\pri$---this law does not hold for output and selection, as they are not prefixes;

    \item
        dual actions leading to a synchronization must have equal priorities (cf.\ Def.\ \labelcref{d:props}).
\end{enumerate}
Judgments are of the form $P \vdash \Omega; \Gamma$, where:
\begin{itemize}
    \item
        $P$ is a process;
    \item
        $\Gamma$ is a context that assigns types to channels ($x: A$);
    \item
        $\Omega$ is a context that assigns tuples of types to recursion variables ($X: (A, B, \ldots)$).
\end{itemize}
A judgment $P \vdash \Omega; \Gamma$ then means that $P$ can be typed in accordance with the type assignments for names recorded in $\Gamma$ and the recursion variables in $\Omega$.
Intuitively, the latter context ensures that the context endpoints between recursive definitions and calls concur.
Both contexts $\Gamma$ and $\Omega$ obey \emph{exchange}: assignments may be silently reordered.
$\Gamma$ is \emph{linear}, disallowing \emph{weakening} (i.e., all assignments must be used) and \emph{contraction} (i.e., assignments may not be duplicated).
$\Omega$ allows weakening and contraction, because a recursive definition may be called \emph{zero or more} times.

The empty context is written $\emptyset$.
In writing $\Gamma, x: A$ (or similarly for $\Omega$) we assume that $x \notin \dom(\Gamma)$.
We write $\lift{t} \Gamma$ to denote the component-wise extension of lift (\Cref{d:lift}) to typing contexts.
Also, we write $\pr(\Gamma)$ to denote the least priority of all types in $\Gamma$ (\Cref{d:priority}).
An assignment $\tilde{z}: \tilde{A}$ means $z_1: A_1, \ldots, z_k: A_k$.

\begin{figure}[t]
    \begin{mdframed}
        {
            \begin{mathpar}
                \inferrule*[right=\rLab{\scc{Empty}}]
                { }
                { \0 \vdash \Omega; \emptyset }
                \and
                \inferrule*[right=\rLab{$\bullet$}]
                { P \vdash \Omega; \Gamma }
                { P \vdash \Omega; \Gamma, x: \bullet }
                \and
                \inferrule*[right=\rLab{\scc{Id}}]
                { }
                { x \fwd y \vdash \Omega; x: \ol{A}, y: A }
                \\
                \inferrule*[right=\rLab{\scc{Mix}}]
                {
                    P \vdash \Omega; \Gamma
                    \\
                    Q \vdash \Omega; \Delta
                }
                { P \| Q \vdash \Omega; \Gamma, \Delta }
                \and
                \inferrule*[right=\rLab{\scc{Cycle}}]
                { P \vdash \Omega; \Gamma, x: A, y: \ol{A} }
                { \nu{x y} P \vdash \Omega; \Gamma }
                \\
                \inferrule*[right=\rLab{$\tensor$}]
                { }
                { x[y,z] \vdash \Omega; x: A \tensor^\pri B, y: \ol{A}, z: \ol{B} }
                \and
                \inferrule*[right=\rLab{$\parr$}]
                {
                    P \vdash \Omega; \Gamma, y: A, z: B
                    \\
                    \pri < \pr(\Gamma)
                }
                { x(y, z) \sdot P \vdash \Omega; \Gamma, x: A \parr^\pri B }
                \and
                \inferrule*[right=\rLab{$\oplus$}]
                { j \in I }
                { x[z] \puts j \vdash \Omega; x: \oplus^\pri\{i: A_i\}_{i \in I}, z: \ol{A_j} }
                \and
                \inferrule*[right=\rLab{$\&$}]
                {
                    \forall i \in I.~ P_i \vdash \Omega; \Gamma, z: A_i
                    \\
                    \pri < \pr(\Gamma)
                }
                { x(z) \gets \{i: P_i\}_{i \in I} \vdash \Omega; \Gamma, x: \&^\pri\{i: A_i\}_{i \in I} }
                \and
                \inferrule*[right=\rLab{\scc{Rec}}]
                {
                    P \vdash \Omega, X: {\tilde{A}}; \tilde{z}: \tilde{U}
                    \\
                    t \in \mbb{N} > \maxpr(\tilde{A})
                    \\
                    \forall U_i \in \tilde{U}.~ U_i = \unfold^t(\mu X \sdot A_i)
                }
                { \mu X(\tilde{z}) \sdot P \vdash \Omega; \tilde{z}: \widetilde{\mu X \sdot A} }
                \and
                \inferrule*[right=\rLab{\scc{Var}}]
                {
                    t \in \mbb{N}
                    \\
                    \forall U_i \in \tilde{U}.~ U_i = \mu X \sdot \lift{t} A_i
                }
                { X\call{\tilde{z}} \vdash \Omega, X: \tilde{A}; \tilde{z}: \tilde{U} }
            \end{mathpar}

            \hbox to \textwidth{\leaders\hbox to 3pt{\hss . \hss}\hfil}

            \vspace{-1.0em}
            \begin{mathpar}
                \inferrule*[right=\rLab{$\tensor^\star$}]
                { P \vdash \Omega; \Gamma, y: A, x: B }
                { \ol{x}[y] \cdot P \vdash \Omega; \Gamma, x: A \tensor^\pri B }
                \and
                \inferrule*[right=\rLab{$\oplus^\star$}]
                {
                    P \vdash \Omega; \Gamma, x: A_j
                    \\
                    j \in I
                }
                { \ol{x} \triangleleft j \cdot P \vdash \Omega; \Gamma, x: \oplus^\pri\{i: A_i\}_{i \in I} }
                \and
                \inferrule*[right=\rLab{\scc{Lift}}]
                {
                    P \vdash \Omega; \Gamma
                    \\
                    t \in \mbb{N}
                }
                { P \vdash \Omega; \lift{t} \Gamma }
            \end{mathpar}
        }%
    \end{mdframed}

    \caption{The typing rules of \APCP (top) and admissible rules (bottom).}
    \label{f:apcpInf}
\end{figure}

\Cref{f:apcpInf} (top) gives the typing rules.
We describe the typing rules from a \emph{bottom-up} perspective.
Rule~\scc{Empty} types an inactive process with no endpoints.
Rule~$\bullet$ silently removes a closed endpoint to the typing context.
Rule~\scc{Id} types forwarding between endpoints of dual type.
Rule~\scc{Mix} types the parallel composition of two processes that do not share assignments on the same endpoints.
Rule~\scc{Cycle} types a restriction, where the two restricted endpoints must be of dual type.
Note that a single application of \scc{Mix} followed by \scc{Cycle} coincides with the usual Rule~\scc{Cut} in type systems based on linear logic~\cite{conf/concur/CairesP10,conf/icfp/Wadler12}.
Rule~$\tensor$ types an output action; this rule does not have premises to provide a continuation process, leaving the free endpoints to be bound to a continuation process using \scc{Mix} and \scc{Cycle}.
Similarly, Rule~$\oplus$ types an unbounded selection action.
Priority checks are confined to Rules~$\parr$ and $\&$, which type an input and a branching prefix, respectively.
In both cases, the used endpoint's priority must be lower than the priorities of the other types in the continuation's typing context.

Rule~\scc{Rec} types a recursive definition by introducing a recursion variable to the recursion context whose tuple of types concurs with the contents of the recursive types in the typing context, where contractivity is guaranteed by requiring that the eliminated recursion variable may not occur unguarded in each of the context's types.
At the same time, the recursive types in the context are unfolded, and their priorities are lifted by a common value, denoted $t$ in the rule, that must be greater than the highest priority occurring in the original types (cf.\ \Cref{d:maxpr}). Using a ``common lifter'', i.e., lifting the priorities of all types by the same amount is crucial: it maintains the relation between the priorities of the types in the context.

Rule~\scc{Var} types a recursive call on a variable in the recursive context.
The rule requires that all the types in the context are recursive on the recursion variable called, and that the types inside the recursive definitions concur with the respective types assigned to the recursion varialbe in the recursive context.
As mentioned before, the types associated to the introduced and consequently eliminated recursion variable is crucial in ensuring that a recursion is called with endpoints of the same type as required by its definition.

The binding of output and selection actions to continuation processes (\Cref{n:sugar}) is derivable in \APCP.
The corresponding typing rules in \Cref{f:apcpInf} (bottom) are admissible using \scc{Mix} and \scc{Cycle}.
Note that it is not necessary to include rules for the sugared input and branching in \Cref{n:sugar}, because they rely on $\alpha$-renaming only.
\Cref{f:apcpInf} (bottom) also includes an admissible Rule~\scc{Lift} that lifts a process' priorities.

The following result assures that, given a type, we can construct a process with an endpoint typable with the given type:

\begin{proposition}\label{p:charProc}
    Given a type $A$, there exists a $P$ such that $P \vdash \Omega; x:A$.
\end{proposition}

\begin{proof}
    We inductively define a function $\chr^x(A)$ that, given a type $A$ and an endpoint $x$, constructs a process that performs the behavior described by $A$:
    \begin{align*}
        \chr^x(A \tensor^\pri B)
        &\deq \ol{x}[y] \cdot (\chr^y(A) \| \chr^x(B))
        &
        \chr^x(\oplus^\pri \{i:A_i\}_{i \in I})
        &\deq \ol{x} \puts j \cdot \chr^x(A_j)
        \quad \text{[any $j \in I$]}
        \\
        \chr^x(A \parr^\pri B)
        &\deq x(y) \sdot (\chr^y(A) \| \chr^x(B))
        &
        \chr^x(\&^\pri \{i:A_i\}_{i \in I})
        &\deq x \gets \{i:\chr^x(A_i)\}_{i \in I}
        \\
        \chr^x(\bullet)
        &\deq \0
        \qquad\qquad
        \chr^x(\mu X \sdot A)
        \deq \mu X(x) \sdot \chr^x(A)
        \qquad\qquad
        \chr^x(X)
        \deq X\call{x}
        \span\span
    \end{align*}
    For finite types,  we have: $\chr^x(A) \vdash \emptyset; x:A$.
    For simplicity, we omit details about recursive types, which require unfolding.
    For closed, recursive types, we have: $\chr^x(\mu X \sdot A) \vdash \emptyset; x: \mu X \sdot A$.
\end{proof}

\paragraph{Type preservation}

Well-typed processes satisfy protocol fidelity, communication safety, and deadlock-freedom.
The first two properties follow directly from \emph{type preservation} (also known as \emph{subject reduction}), which ensures that reduction preserves typing.
In contrast to Caires and Pfenning~\cite{conf/concur/CairesP10} and Wadler~\cite{conf/icfp/Wadler12}, where type preservation corresponds to the elimination of (top-level) applications of rule \scc{Cut}, in \APCP it corresponds to the more general elimination of (top-level) applications of rule \scc{Cycle}.

Because reduction is closed under structural congruence, type preservation relies on \emph{subject congruence}: structural congruence preserves typing.
The structural congruence rule that unfolds recursive definitions requires care, because the types of the unfolded process are also unfolded.
Hence, type preservation holds \emph{up to unfolding}.
We formalize this with the relation $(P; \Gamma) \tRel \Gamma'$, which denotes that $\Gamma$ and $\Gamma'$ are equal up to (un)folding of recursive types consistent with the typing of $P$ under $\Gamma$:

\begin{definition}
\label{d:typingUnfold}
    We define an asymmetrical relation between a process / typing pair $(P; \Gamma)$ and a typing context $\Gamma'$, denoted $(P; \Gamma) \tRel \Gamma'$.
    The relation is defined inductively, with the following being the most important rules relating unfolded and non-unfolded types:
    \begin{mathpar}
        \inferrule*[right=\rLab{$\scc{Fold}$}]
        {
            ( P \big\{ (\mu X(\tilde{y}) \sdot P \{\tilde{y} / \tilde{z}\}) / X\<\tilde{y}\> \big\}; \tilde{z}:\tilde{U} ) \tRel \tilde{z}:\tilde{U'}
            \\
            \forall U'_i \in \tilde{U'}.~ U'_i = \unfold^t(\mu X \sdot A'_i)
        }
        { ( P \big\{ (\mu X(\tilde{y}) \sdot P \{\tilde{y} / \tilde{z}\}) / X\<\tilde{y}\> \big\}; \tilde{z}:\tilde{U} ) \tRel \tilde{z}:\widetilde{\mu X \sdot A'} }
        \and
        \inferrule*[right=\rLab{$\scc{Unfold}$}]
        {
            ( \mu X(\tilde{z}) \sdot P; \tilde{z}:\widetilde{\mu X \sdot A}) \tRel \tilde{z}:\widetilde{\mu X \sdot A'}
            \\
            \forall U'_i \in \tilde{U'}.~ U'_i = \unfold^t(\mu X \sdot A'_i)
        }
        { ( \mu X(\tilde{z}) \sdot P; \tilde{z}:\widetilde{\mu X \sdot A}) \tRel \tilde{z}:\tilde{U'} }
    \end{mathpar}
    The other rules follow the typing rules in \Cref{f:apcpInf}; the following is a selection of rules:
    \begin{mathpar}
        \inferrule*[right=\rLab{$\scc{Empty}$}]
        { }
        { (\0; \emptyset) \tRel \emptyset }
        \and
        \inferrule*[right=\rLab{$\scc{Mix}$}]
        { (P; \Gamma) \tRel \Gamma' \\ (Q; \Delta) \tRel \Delta' }
        { (P \| Q; \Gamma, \Delta) \tRel \Gamma', \Delta' }
        \and
        \inferrule*[right=\rLab{$\parr$}]
        { (P; \Gamma, y:A, z:B) \tRel \Gamma', y:A', z:B' }
        { (x(y,z) \sdot P; \Gamma, x:A \parr^\pri B) \tRel \Gamma', x:A' \parr^\pri B' }
    \end{mathpar}

    \noindent
    We write $(P; \Gamma) \tbRel (Q; \Gamma')$ if $(P; \Gamma) \tRel \Gamma'$ and $(Q; \Gamma') \tRel \Gamma$.
\end{definition}

\begin{theorem}[Type Preservation]\label{t:subjectRed}
    If $P \vdash \Omega; \Gamma$ and $P \redd Q$, then there exists $\Gamma'$ such that $Q \vdash \Omega; \Gamma'$ and ${(P; \Gamma) \tbRel (Q; \Gamma')}$.
\end{theorem}

\paragraph{Deadlock-freedom}

The deadlock-freedom result for \APCP adapts that for \PCP~\cite{conf/fossacs/DardhaG18}.
As mentioned before, binding  asynchronous outputs and selections to continuations involves additional, low-level uses of \scc{Cycle}, which we cannot eliminate through process reduction.
Therefore, top-level deadlock-freedom (referred to as \emph{progress}) holds for \emph{live processes} (\Cref{t:progress}).
A process is live if it is equivalent to a restriction on \emph{active names} that perform unguarded actions.
This way, e.g., in $x[y,z]$ the name $x$ is active, but $y$ and $z$ are not.
We additionally need a notion of \emph{evaluation context}, under which reducible forwarders may occur.

\begin{definition}[Active Names]\label{d:an}
    The \emph{set of active names} of $P$, denoted $\an(P)$, contains the (free) names that are used for unguarded actions (output, input, selection, branching):
    \begin{align*}
        \an(x[y,z])
        & \deq \{x\}
        &
        \an(x(y,z) \sdot P)
        & \deq \{x\}
        &
        \an(\0)
        & \deq \emptyset
        \\
        \an(x[z] \triangleleft j)
        & \deq \{x\}
        &
        \an(x(z) \triangleright \{i: P_i\}_{i \in I})
        & \deq \{x\}
        &
        \an(x \fwd y)
        & \deq \emptyset
        \\
        \an(P \| Q)
        & \deq \an(P) \cup \an(Q)
        &
        \an(\mu X(\tilde{x}) \sdot P)
        & \deq \an(P)
        \\
        \an(\nu{x y}P)
        & \deq \an(P) \setminus \{x,y\}
        &
        \an(X\call{\tilde{x}})
        & \deq \emptyset
    \end{align*}
\end{definition}

\begin{definition}[Evaluation Context]
\label{d:evalcontext}
    Evaluation contexts $\mcl{E}$ are defined by the following grammar:
    \begin{align*}
        \mcl{E} &\deq [] \sepr \mcl{E} \| P \sepr \nu{xy}\mcl{E} \sepr \mu X(\tilde{x}) \sdot \mcl{E}
    \end{align*}
    We write $\mcl{E}[P]$ to denote the process obtained by replacing the hole $[]$ in $\mcl{E}$ by $P$.
\end{definition}

\begin{definition}[Live Process]\label{d:live}
    A process $P$ is \emph{live}, denoted $\live(P)$, if
    \begin{enumerate}
        \item there are names $x,y$ and process $P'$ such that $P \equiv \nu{x y}P'$ with $x,y \in \an(P')$, or
        \item there are names $x,y,z$ and process $P'$ such that $P \equiv \mcl{E}[\nu{yz}(x \fwd y \| P')]$.
    \end{enumerate}
\end{definition}

We additionally need to account for recursion: as recursive definitions do not entail reductions, we must fully unfold them before eliminating \scc{Cycle}s:

\begin{lemma}[Unfolding]\label{l:unfold}
    If $P \vdash \Omega; \Gamma$, then there is a process $P^\star$ such that $P^\star \equiv P$ and $P^\star$ is not of the form $\mu X(\tilde{z}) \sdot Q$.
\end{lemma}

\noindent
Deadlock-freedom, given next, states that typable processes that are live can reduce.
It follows from an analysis of the priorities in the typing of the process, which makes it possible to find a pair of non-blocked, parallel, dual actions on connected endpoints, such that a communication can occur.
The analysis also considers the possibility that a blocking action is on an endpoint which is not connected (i.e., the endpoint is free), in which case a commuting conversion can be performed.
Confer the full proof by Van den Heuvel and Pérez~\cite[Theorem~7]{report/vdHeuvelP21B} for more details.

\begin{theorem}[Progress]\label{t:progress}
    If $P \vdash \emptyset; \Gamma$ and $\live(P)$, then there is a process $Q$ such that $P \redd Q$.
\end{theorem}

\noindent
We now state the deadlock-freedom result formalized by Van den Heuvel and Pérez~\cite{report/vdHeuvelP21B}.
Following, e.g., Caires and Pfenning~\cite{conf/concur/CairesP10} and Dardha and Gay~\cite{conf/fossacs/DardhaG18}, it concerns processes typable under empty contexts.

\begin{theorem}[Deadlock-freedom~\cite{report/vdHeuvelP21B}]\label{t:dlFree}
    If $P \vdash \emptyset; \emptyset$, then either $P \equiv \0$ or $P \redd Q$ for some $Q$.
\end{theorem}

\paragraph{Fairness}

Processes typable under empty contexts are not only deadlock-free, they are \emph{fair}: for each endpoint in the process, we can eventually observe a reduction involving that endpoint.
To formalize this property, we define \emph{labeled reductions}, which expose details about a communication:

\begin{definition}[Labeled Reductions]\label{d:procLred}
    Consider the labels
    \begin{align*}
        \alpha \deq x \fwd y ~\sepr~ x \rangle y:a ~\sepr~ x \rangle y:\ell
        \qquad\qquad
        \text{(forwarding, output/input, selection/branching)}
    \end{align*}
    where each label has subjects $x$ and $y$.
    The \emph{labeled reduction} $P \lredd{\alpha} Q$ is defined by the following rules:
    \begin{align*}
        \nu{y z}( x \fwd y \| P ) \lredd{x \fwd y} P \subst{x/z}
        \qquad
        \nu{x y}( x[a,b] \| y(v,z) \sdot P ) \lredd{x \rangle y:a} P \subst{a/v, b/z}
        \\
        \nu{x y}( x[b] \puts j \| y(z) \gets \{i: P_i\}_{i \in I} ) \lredd{x \rangle y:j} P_j \subst{b/z} \quad \text{(if $j \in I$)}
    \end{align*}
    \begin{mathpar}
        \inferrule
        {
            P \equiv P'
            \\
            P' \lredd{\alpha} Q'
            \\
            Q' \equiv Q
        }
        { P \lredd{\alpha} Q }
        \and
        \inferrule
        { P \lredd{\alpha} Q }
        { \nu{x y} P \lredd{\alpha} \nu{x y} Q }
        \and
        \inferrule
        { P \lredd{\alpha} Q }
        { P \| R \lredd{\alpha} Q \| R }
    \end{mathpar}
\end{definition}

\begin{proposition}
    For any $P$ and $P'$, $P \redd P'$ if and only if there exists a label $\alpha$ such that $P \lredd{\alpha} P'$.
\end{proposition}

\begin{proof}
    Immediate by definition, for each reduction in \Cref{f:procdef} (bottom) corresponds to a labeled reduction, and vice versa.
\end{proof}

Fairness states that processes typable under empty contexts have at least one finite reduction sequence ($\redd^\star$) that enables a labeled reduction involving a \emph{pending} endpoint---an endpoint that occurs as the subject of an action, and is not bound by input or branching (see below).
Clearly, the typed process may have other reduction sequences, not necessarily finite.
Confer the full proof by Van den Heuvel and P\'erez~\cite[Theorem~11]{report/vdHeuvelP21B} for more details.

\begin{definition}[Pending Names]\label{d:pending}
    Given a process $P$, we define the set of \emph{pending names} of $P$, denoted $\pn(P)$, as follows:
    \begin{align*}
        \pn(x[y,z])
        &\deq \{x\}
        &
        \pn(x(y,z).P)
        &\deq \{x\} \cup (\pn(P) \setminus \{y,z\})
        &
        \pn(\0)
        &\deq \emptyset
        \\
        \pn(x[z] \puts j)
        &\deq \{x\}
        &
        \pn(x(z) \gets \{i: P_i\}_{i \in I})
        &\deq \{x\} \cup (\bigcup_{i \in I} \pn(P_i) \setminus \{z\})
        &
        \pn(x \fwd y)
        &\deq \{x,y\}
        \\
        \pn(P \| Q)
        &\deq \pn(P) \cup \pn(Q)
        &
        \pn(\mu X(\tilde{x}) \sdot P)
        &\deq \pn(P)
        \\
        \pn(\nu{x y}P)
        &\deq \pn(P)
        &
        \pn(X \call{\tilde{x}})
        &\deq \emptyset
    \end{align*}
\end{definition}

\begin{theorem}[Fairness~\cite{report/vdHeuvelP21B}]\label{t:fairness}
    Suppose given a process $P \vdash \emptyset; \emptyset$.
    Then, for every $x \in \pn(P)$ there exists a process $P'$ such that $P \redd^\star P'$ and $P' \lredd{\alpha}\, Q$, for some process $Q$ and label $\alpha$ with subject $x$.
\end{theorem}

\subsection*{Examples}

To illustrate \APCP processes and their session types, we give implementations for the three participants in $\gtc{G_{\sff{auth}}}$ in \Cref{s:mpstAPCP:intro}.

\begin{example}
    \label{x:mpstAPCP:impl}
    Processes $P$, $Q$, and $R$ are typed implementations for participants $c$, $s$, and $a$, respectively, each usng ia single channel to perform the exchanges described by~$\gtc{G_{\sff{auth}}}$.
    \begin{align*}
        \setlength\arraycolsep{.2ex}
        \emptyset &\vdash P
        \deq \pRes X(\ci{c}{\mu}) ; \pBra* \ci{c}{\mu} > {
            \left\{
                \begin{array}{@{}l@{}l@{}}
                    \sff{login}
                    & {}: \pIn \ci{c}{\mu}(u) ; \pSel* \ci{c}{\mu} < \sff{passwd} \cdot \pOut* \ci{c}{\mu}[\bm{logmein345}] \cdot \pCall X<\ci{c}{\mu}> ,
                    \\
                    \sff{quit}
                    & {}: \pIn \ci{c}{\mu}(w) ; \pSel* \ci{c}{\mu} < \sff{quit} \cdot \pOut* \ci{c}{\mu}[z] \cdot \0
                \end{array}
            \right\}
        }
        \\
        &\typInf \ci{c}{\mu}:\mu X . \&^{2} \left\{
            \begin{array}{@{}l@{}l@{}}
                \sff{login}
                & {}: \bullet \parr^{3} \oplus^{4} \{ \sff{passwd} : \bullet \tensor^{5} X \} ,
                \\
                \sff{quit}
                & {}: \bullet \parr^{3} \oplus^{4} \{ \sff{quit} : \bullet \tensor^{5} \bullet \}
            \end{array}
        \right\}
        \\[5pt]
        \emptyset &\vdash Q
        \deq \pRec X(\ci{s}{\mu}) ; \pSel* \ci{s}{\mu} < \sff{login} \cdot \pOut* \ci{s}{\mu}[u] \cdot \pBra* \ci{s}{\mu} > {\{ \sff{auth} : \pIn \ci{s}{\mu}(v) ; \pCall X<\ci{s}{\mu}> \}}
        \\
        &\typInf \ci{s}{\mu}:\mu X . \oplus^{0} \{ \sff{login} : \bullet \tensor^{1} \&^{10} \{ \sff{auth} : \bullet \parr^{11} X \} , \sff{quit} : \bullet \tensor^{1} \bullet \}
        \\[5pt]
        \emptyset &\vdash R
        \deq \pRec X(\ci{a}{\mu}) ; \pBra* \ci{a}{\mu} > {
            \left\{
                \begin{array}{@{}l@{}l@{}}
                    \sff{login}
                    & {}: \pBra* \ci{a}{\mu} > {\{ \sff{passwd} : \pIn \ci{a}{\mu}(u) ; \pSel* \ci{a}{\mu} < \sff{auth} \cdot \pOut* \ci{a}{\mu}[v] \cdot \pCall X<\ci{a}{\mu}> \}} ,
                    \\
                    \sff{quit}
                    & {}: \pBra* \ci{a}{\mu} > {\{ \sff{quit} : \pIn \ci{a}{\mu}(w) ; \0 \}}
                \end{array}
            \right\}
        }
        \\
        &\typInf \ci{a}{\mu}:\mu X . \&^{2} \left\{
            \begin{array}{@{}l@{}l@{}}
                \sff{login}
                & {}: \&^{6} \{ \sff{passwd} : \bullet \parr^{7} \oplus^{8} \{ \sff{auth} : \bullet \tensor^{9} X \} \} ,
                \\
                \sff{quit}
                & {}: \&^{6} \{ \sff{quit} : \bullet \parr^{7} \bullet \}
            \end{array}
        \right\}
    \end{align*}
    Process $P$ is a specific implementation for $c$, where we use `$\bm{logmein345}$' to denote a closed endpoint representing the password string ``logmein345''.
    Similarly, $Q$ is a specific implementation for $s$ that continuously chooses the \sff{login} branch.
\end{example}

Note that the processes above cannot be directly connected to each other to implement $G_\sff{auth}$. Our goal is to enable the composition of (typed) implementations such as  $P$, $Q$, and $R$ in a correct and deadlock-free manner. We shall proceed as follows.
After setting up the routers that enable the composition of these processes according to $G_\sff{auth}$ (\Cref{s:routers}), we will return to this example in \Cref{s:routersInAction}. At that point, it will become clear that the priorities in the types of $P$, $Q$, and $R$ were chosen to ensure the correct composition  with their respective routers.

\section{Global Types and Relative Projection}
\label{s:mpst}

\ifmpstAPCP\else
We analyze multiparty protocols specified as \emph{global types}.
We consider a standard syntax, with session delegation and recursion, subsuming the one given in the seminal paper by Honda \etal~\cite{journal/acm/HondaYC16}.
\fi
In the following, we write $p, q, r, s, \ldots$ to denote \emph{(protocol) participants}.

\begin{definition}[Global Types]
    \ifmpstAPCP \label{d:mpstAPCP:globalTypes}
    \else \label{d:relProj:globalTypes} \fi
    \emph{Global types} $\gtc{G}$ and \emph{message types} $S,T$ are defined as:
    \ifmpstAPCP
    \begin{align*}
        \gtc{G} &
        \begin{array}[t]{@{}l@{}lr@{\kern2em}l@{\kern1em}lr@{}}
            {} \deq {} &
            p {!} q \gtBraces{ \msg i<S> \gtc. \gtc{G} }_{i \in I} & \text{exchange}
            & \sepr &
            \tEnd & \text{end}
            \\ \sepr* &
            \mu X \gtc. \gtc{G} \sepr X & \text{recursion}
            & \sepr &
            \tSkip \gtc. \gtc{G} & \text{skip}
        \end{array}
        \\
        S,T &
        \begin{array}[t]{@{}l@{}lr@{\kern2em}l@{\kern1em}lr@{\kern2em}l@{\kern1em}lr@{}}
            {} \deq {} &
            {!}T . S & \text{send}
            & \sepr &
            \oplus \{ i : S \}_{i \in I} & \text{selection}
            & \sepr &
            \tEnd & \text{end}
            \\ \sepr* &
            {?}T . S & \text{receive}
            & \sepr &
            \& \{ i : S \}_{i \in I} & \text{branch}
        \end{array}
    \end{align*}
    \else
    \begin{align*}
        \gtc{G} &
        \begin{array}[t]{@{}l@{}lr@{\kern2em}l@{\kern1em}lr@{}}
            {} \deq {} &
            p {!} q \gtBraces{ \msg i<S> \gtc. \gtc{G} }_{i \in I} & \text{exchange}
            & \sepr &
            \tEnd & \text{end}
            \\ \sepr* &
            \mu X \gtc. \gtc{G} \sepr X & \text{recursion}
        \end{array}
        \\
        S,T &
        \begin{array}[t]{@{}l@{}lr@{\kern2em}l@{\kern1em}lr@{\kern2em}l@{\kern1em}lr@{}}
            {} \deq {} &
            {!}T . S & \text{send}
            & \sepr &
            \oplus \{ i : S \}_{i \in I} & \text{selection}
            & \sepr &
            \tEnd & \text{end}
            \\ \sepr* &
            {?}T . S & \text{receive}
            & \sepr &
            \& \{ i : S \}_{i \in I} & \text{branch}
        \end{array}
    \end{align*}
    \fi
    We include basic types (e.g., \sff{unit}, \sff{bool}, \sff{int}), which are all syntactic sugar for $\tEnd$.
\end{definition}

\noindent
The type $p {!} q \gtBraces{ \msg i<S_i> \gtc. \gtc{G_i} }_{i \in I}$ specifies a direct exchange from participant $p$ to participant~$q$, which precedes protocol $\gtc{G_i}$:
$p$ chooses a label $i \in I$ and sends it to $q$ along with a message of type $S_i$.
Message exchange is \emph{asynchronous}: the protocol can continue as $\gtc{G_i}$ before the message has been received by $q$.
The type $\mu X \gtc. \gtc{G}$ defines a recursive protocol: whenever a path of exchanges in $\gtc{G}$ reaches the recursion variable $X$, the protocol continues as $\mu X \gtc. \gtc{G}$.
The type $\tEnd$ denotes the completed protocol.
\ifmpstAPCP
For technical convenience, we introduce the construct $\tSkip \gtc. \gtc{G}$, which denotes an unobservable step that precedes $\gtc{G}$.
\fi

Recursive definitions bind recursion variables, so recursion variables not bound by a recursive definition are free.
We write $\frv(\gtc{G})$ to denote the set of free recursion variables of~$\gtc{G}$, and say $\gtc{G}$ is \emph{closed} if $\frv(\gtc{G}) = \emptyset$.
Recursion in global types is tail-recursive and \emph{contractive} (i.e.\ they contain no subexpressions of the form $\mu X_1 \ldots \mu X_n \gtc. X_1$).
\ifappendix\ifmpstAPCP As for the session types in \Cref{s:mpstAPCP:APCP}, we \else We \fi\else We \fi define the unfolding of a recursive global type by substituting copies of the recursive definition for recursive calls, i.e., $\mu X \gtc. \gtc{G}$ unfolds to $\gtc{G} \{ \mu X \gtc. \gtc{G} / X \}$.

\ifmpstAPCP\else
In approaches based on \MPST, the grammar of global types specifies multiparty protocols but does not ensure their correct implementability; such guarantees are given in terms of \emph{well-formedness}, defined as projectability onto all participants (cf.\ \Cref{s:relProj:local}).
\fi

Message types $S,T$ define binary protocols\ifmpstAPCP, not to be confused with the types \ifappendix in \secref{s:mpstAPCP:APCP} \else of \APCP from \Cref{c:APCP}\fi\fi.
Type ${!}T . S$ (resp.\ ${?}T . S$) denotes the send (resp.\ receive) of a message of type $T$ followed by the continuation $S$.
Type $\oplus \{ i : S_i \}_{i \in I}$ denotes \emph{selection}: the output of a choice for a label $i \in I$ followed by the continuation $S_i$.
Type~\mbox{$\& \{ i : S_i \}_{i \in I}$} denotes \emph{branching}: the input of a label $i \in I$ followed by the continuation $S_i$.
Type~$\tEnd$ denotes the end of the protocol.
Due to the tail-recursiveness of session and global types, there are no recursive message types.

It is useful to obtain the set of participants of a global type:

\begin{definition}[Participants of Global Types]
    \ifmpstAPCP \label{d:mpstAPCP:global:part}
    \else \label{d:relProj:global:part} \fi
    We define the \emph{set of participants} of global type~$\gtc{G}$, denoted $\prt(\gtc{G})$:
    \ifmpstAPCP
    \begin{align*}
        \prt(p {!} q \gtBraces{ \msg i<S_i> \gtc. \gtc{G_i} }_{i \in I})
        &\deq \{p,q\} \cup \bigcup_{i \in I} \prt(\gtc{G_i})
        &
        \prt(\tSkip \gtc. \gtc{G})
        &\deq \prt(\gtc{G})
        &
        \prt(\tEnd)
        &\deq \emptyset
        \\
        &
        &
        \prt(\mu X \gtc. \gtc{G})
        &\deq \prt(\gtc{G})
        &
        \prt(X)
        &\deq \emptyset
    \end{align*}
    \else
    \begin{align*}
        \prt(p {!} q \gtBraces{ \msg i<S_i> \gtc. \gtc{G_i} }_{i \in I})
        &\deq \{p,q\} \cup \bigcup_{i \in I} \prt(\gtc{G_i})
        &
        \prt(\tEnd)
        &\deq \emptyset
        \\
        \prt(\mu X \gtc. \gtc{G})
        &\deq \prt(\gtc{G})
        &
        \prt(X)
        &\deq \emptyset
    \end{align*}
    \fi
\end{definition}

\subsection{Relative Types}
\label{ss:reltypes}

While a global type such as $\gtc{G_{\sff{auth}}}$~\ifmpstAPCP\eqref{eq:mpstAPCP:Gauth} \else\eqref{eq:relProj:Gauth} \fi describes a protocol from a vantage point, \emph{relative types} describe the interactions between \emph{pairs} of participants.
\ifmpstAPCP\else
This way, relative types capture the peer-to-peer nature of multiparty protocols.
We develop projection from global types onto relative types (\Cref{s:relProj:relProj}) and use it to establish a new class of \emph{well-formed} global types.
\fi

A choice between participants in a global type is \emph{non-local} if it influences future exchanges between other participants.
Our approach uses \emph{dependencies} to expose these non-local choices
in the relative types of these other participants.

\begin{definition}[Relative Types]
    \ifmpstAPCP \label{d:mpstAPCP:relTypes}
    \else \label{d:relProj:relTypes} \fi
    \emph{Relative types} $R$ are defined as follows, where the $S_i$ are message types\ifmpstAPCP\else\ (cf.\ \Cref{d:relProj:globalTypes})\fi:
    \ifmpstAPCP
    \[
        \rtc{R}
        \begin{array}[t]{@{}l@{}lr@{\kern2em}l@{\kern1em}lr@{}}
            {} \deq {} &
            p {!} q \rtBraces{ \msg i<S_i> \rtc. \rtc{R} }_{i \in I} & \text{exchange}
            & \sepr &
            \tEnd & \text{end}
            \\ \sepr* &
            (p {!} r) {!} q \rtBraces{ i \rtc. \rtc{R} }_{i \in I} & \text{output dependency}
            & \sepr &
            \mu X \rtc. \rtc{R} \sepr X & \text{recursion}
            \\ \sepr* &
            (p {?} r) {!} q \rtBraces{ i \rtc. \rtc{R} }_{i \in I} & \text{input dependency}
            & \sepr &
            \tSkip \rtc. \rtc{R} & \text{skip}
        \end{array}
    \]
    \else
    \[
        \rtc{R}
        \begin{array}[t]{@{}l@{}lr@{\kern2em}l@{\kern1em}lr@{}}
            {} \deq {} &
            p {!} q \rtBraces{ \msg i<S> \rtc. \rtc{R} }_{i \in I} & \text{exchange}
            & \sepr &
            \tEnd & \text{end}
            \\ \sepr* &
            (p {!} r) {!} q \rtBraces{ i \rtc. \rtc{R} }_{i \in I} & \text{output dependency}
            & \sepr &
            \mu X \rtc. \rtc{R} \sepr X & \text{recursion}
            \\ \sepr* &
            (p {?} r) {!} q \rtBraces{ i \rtc. \rtc{R} }_{i \in I} & \text{input dependency}
        \end{array}
    \]
    \fi
\end{definition}

\noindent
We detail the syntax above.
\begin{itemize}
    \item
        Type $p {!} q \rtBraces{ \msg i<S_i> \rtc. \rtc{R_i} }_{i \in I}$ specifies that $p$ must choose a label $i \in I$ and send it to $q$ along with a message of type $S_i$ after which the protocol continues with $\rtc{R_i}$.
    \item
        Given a participant $r$ different from $q,r$, type $(p {?} r) {!} q \rtBraces{ i \rtc. \rtc{R_i} }_{i \in I}$ expresses a dependency: a non-local choice between $p$ and $r$ which influences the protocol between $p$ and $q$.
        Here, the dependency indicates that after $p$ receives from $r$ the chosen label, $p$ must forward it to $q$, determining the protocol between $p$ and $q$.
    \item
        Similarly, type $(p {!} r) {!} q \rtBraces{ i \rtc. \rtc{R_i} }_{i \in I}$ expresses a dependency, which indicates that after $p$ sends to $r$ the chosen label, $p$ must forward it to $q$.
    \item
        Types $\mu X \rtc. \rtc{R}$ and $X$ define recursion, just as their global counterparts.
    \item
        The type $\tEnd$ specifies the end of the protocol between $p$ and $q$.
    \ifmpstAPCP
    \item
        The type $\tSkip \rtc. \rtc{R}$ denotes an unobservable step that precedes $\rtc{R}$.
    \fi
\end{itemize}

\begin{definition}[Participants of Relative Types]
    \ifmpstAPCP \label{d:mpstAPCP:relative:part}
    \else \label{d:relProj:local:part} \fi
    We define the \emph{set of participants} of relative type $\rtc{R}$, denoted $\prt(\rtc{R})$:
    \ifmpstAPCP
    \begin{align*}
        \prt(p {!} q \rtBraces{ \msg i<S_i> \rtc. \rtc{R_i} }_{i \in I})
        &\deq \{p,q\} \cup \bigcup_{i \in I} \prt(\rtc{R_i})
        &
        \prt(\tSkip \rtc. \rtc{R})
        &\deq \prt(\rtc{R})
        &
        \prt(\tEnd)
        &\deq \emptyset
        \\
        \prt((p {?} r) {!} q \rtBraces{ i \rtc. \rtc{R_i} }_{i \in I})
        &\deq \{p,q\} \cup \bigcup_{i \in I} \prt(\rtc{R_i})
        &
        \prt(\mu X \rtc. \rtc{R})
        &\deq \prt(\rtc{R})
        &
        \prt(X)
        &\deq \emptyset
        \\
        \prt((p {!} r) {!} q \rtBraces{ i \rtc. \rtc{R_i} }_{i \in I})
        &\deq \{p,q\} \cup \bigcup_{i \in I} \prt(\rtc{R_i})
    \end{align*}
    \else
    \begin{align*}
        \prt(p {!} q \rtBraces{ \msg i<S_i> \rtc. \rtc{R_i} }_{i \in I})
        &\deq \{p,q\} \cup \bigcup_{i \in I} \prt(\rtc{R_i})
        &
        \prt(\tEnd)
        &\deq \emptyset
        \\
        \prt((p {?} r) {!} q \rtBraces{ i \rtc. \rtc{R_i} }_{i \in I})
        &\deq \{p,q\} \cup \bigcup_{i \in I} \prt(\rtc{R_i})
        &
        \prt(\mu X \rtc. \rtc{R})
        &\deq \prt(\rtc{R})
        &
        \prt(X)
        &\deq \emptyset
        \\
        \prt((p {!} r) {!} q \rtBraces{ i \rtc. \rtc{R_i} }_{i \in I})
        &\deq \{p,q\} \cup \bigcup_{i \in I} \prt(\rtc{R_i})
    \end{align*}
    \fi
\end{definition}

We introduce some useful notation:

\begin{notation}\leavevmode
    \begin{itemize}
        \item
            We write $p {!} q (\msg i<S>) \gtc. \gtc{G}$ for a global type with a single branch $p {!} q \gtBraces{ \msg i<S> \gtc. \gtc{G} }$ (and similarly for exchanges and dependencies in relative types).
        \item
            We omit \sff{unit} message types from global and relative types, writing, e.g., $i \gtc. \gtc{G}$ for $\msg i<\sff{unit}> \gtc. \gtc{G}$.
        \ifmpstAPCP
        \item
            Given $k > 1$, we write $\tSkip^k$ for a sequence of $k$ $\tSkip$s.
        \fi
    \end{itemize}
\end{notation}

\subsection{Relative Projection and Well-Formedness}\label{ss:relproj}

\begin{figure}[t]
    \begin{mdframed}
        \begin{align*}
            \span \detdep((p,q), s {!} r \{i\<S_i\> \sdot G_i\}_{i \in I}) \deq \begin{cases}
                \tSkip \sdot (G_{i'} \wrt (p,q)) ~\text{[any $i' \in I$]}
                & \text{if $\forall i,j.$}
                \\
                & \text{$G_{i} \wrt (p,q) = G_{j} \wrt (p,q)$}
                \\
                p{!}r \{i \sdot (G_i \wrt (p,q))\}_{i \in I}
                & \text{if $p=s$}
                \\
                q{!}r \{i \sdot (G_i \wrt (p,q))\}_{i \in I}
                & \text{if $q=s$}
                \\
                p{?}s \{i \sdot (G_i \wrt (p,q))\}_{i \in I}
                & \text{if $p=r$}
                \\
                q{?}s \{i \sdot (G_i \wrt (p,q))\}_{i \in I}
                & \text{if $q=r$}
            \end{cases}
        \end{align*}

        \hbox to \textwidth{\leaders\hbox to 3pt{\hss . \hss}\hfil}

        \begin{align*}
            (s {!} r \{i\<S_i\> \sdot G_i\}_{i \in I}) \wrt (p,q)
            &\deq \begin{cases}
                p \{i\<S_i\> \sdot (G_i \wrt (p,q))\}_{i \in I}
                & \text{if $p=s$ and $q=r$}
                \\
                q \{i\<S_i\> \sdot (G_i \wrt (p,q))\}_{i \in I}
                & \text{if $q=s$ and $p=r$}
                \\
                \detdep((p,q), s {!} r \{i\<S_i\> \sdot G_i\}_{i \in I})
                & \text{otherwise}
            \end{cases}
            \displaybreak[1] \\
            (\mu X \sdot G) \wrt (p,q)
            &\deq \begin{cases}
                \mu X \sdot (G \wrt (p,q))
                & \text{if $G \wrt (p,q)$ defined and contractive on $X$}
                \\
                \bullet
                & \text{otherwise}
            \end{cases}
            \displaybreak[1] \\
            \span
            X \wrt (p,q)
            \deq X
            \qquad
            \bullet \wrt (p,q)
            \deq \bullet
            \qquad
            (\tSkip \sdot G) \wrt (p,q)
            \deq \tSkip \sdot (G \wrt (p,q))
        \end{align*}
        Above, $\tSkip^\ast$ denotes a sequence of zero or more $\tSkip$.
    \end{mdframed}

    \caption{
        Dependency Detection (top), and Relative Projection (bottom, cf.\ \Cref{d:relproj}).
        \\
        When a side-condition does not hold, either is undefined.
    }
    \label{f:relproj}
\end{figure}

We define \emph{relative projection} for global types.
We want relative projection to fail when it would return a non-contractive recursive type: hence, we define contractiveness on relative types:

\begin{definition}[Contractive Relative Types]
    \ifmpstAPCP \label{d:mpstAPCP:relProj:contr}
    \else \label{d:relProj:contr} \fi
    Given a relative type $\rtc{R}$ and a recursion variable $X$, we say $\rtc{R}$ \emph{is contractive on $X$} if either of the following holds:
    \begin{itemize}
        \item
            $\rtc{R}$ contains an exchange, or
        \item
            $\rtc{R}$ ends in a recursive call on any $Y \neq X$.
    \end{itemize}
\end{definition}

Relative projection then relies on the contractiveness of relative types.
It also relies on an auxiliary function to determine if a dependency message is needed and possible.

\begin{algorithm}[t]
    \DontPrintSemicolon
    \SetAlgoNoEnd
    \SetInd{.2em}{.7em}
    \SetSideCommentRight
    \Def{$\gtc{G} \wrt (p,q)$}{

        \Switch{$\gtc{G}$}{

            \uCase{$s {!} r \gtBraces{ \msg i<T_i> \gtc. \gtc{G_i} }_{i \in I}$}{
                $\forall i \in I.~ \rtc{R_i} \deq \gtc{G_i} \wrt (p,q)$ \;

                \lIf{$(p = s \wedge q = r)$}{
                    \KwRet
                    $p {!} q \rtBraces{ \msg i<T_i> \rtc. \rtc{R_i} }_{i \in I}$
                }
                \ifmpstAPCP \label{li:mpstAPCP:rpEx1}
                \else \label{li:relProj:rpEx1} \fi

                \lElseIf{$(q = s \wedge p = r)$}{
                    \KwRet
                    $q {!} p \rtBraces{ \msg i<T_i> \rtc. \rtc{R_i} }_{i \in I}$
                }
                \ifmpstAPCP \label{li:mpstAPCP:rpEx2}
                \else \label{li:relProj:rpEx2} \fi

                \lElse{
                    \KwRet $\ddep((p,q),\gtc{G})$
                }
                \ifmpstAPCP \label{li:mpstAPCP:rpDdep}
                \else \label{li:relProj:rpDdep} \fi
            }
            \ifmpstAPCP \label{li:mpstAPCP:rpEx}
            \else \label{li:relProj:rpEx} \fi

            \uCase{$\mu X \gtc. \gtc{G'}$}{
                $\rtc{R'} \deq \gtc{G'} \wrt (p,q)$ \;

                \lIf{$(\text{$\rtc{R'}$ defined and contractive on $X$})$}{
                    \KwRet
                    $\mu X \rtc. \rtc{R'}$
                }
                \ifmpstAPCP \label{li:mpstAPCP:rpRecRec}
                \else \label{li:relProj:rpRecRec} \fi

                \lElse{
                    \KwRet
                    $\tEnd$
                }
                \ifmpstAPCP \label{li:mpstAPCP:rpRecEnd}
                \else \label{li:relProj:rpRecEnd} \fi
            }
            \ifmpstAPCP \label{li:mpstAPCP:rpRec}
            \else \label{li:relProj:rpRec} \fi

            \lCase{$X$}{
                \KwRet
                $X$
            }
            \ifmpstAPCP \label{li:mpstAPCP:rpCall}
            \else \label{li:relProj:rpCall} \fi

            \lCase{$\tEnd$}{
                \KwRet
                $\tEnd$
            }
            \ifmpstAPCP \label{li:mpstAPCP:rpEnd}
            \else \label{li:relProj:rpEnd} \fi

            \ifmpstAPCP
            \lCase{$\tSkip \gtc. \gtc{G'}$}{
                \KwRet
                $\tSkip \rtc. (\gtc{G'} \wrt (p,q))$
            }\label{li:mpstAPCP:rpSkip}
            \fi

        }

    }

    \caption{\ifmpstAPCP\nameref{d:mpstAPCP:relProj} \else\nameref{d:relProj:relProj} \fi of $\gtc{G}$ onto $(p,q)$ (\ifmpstAPCP\Cref{d:mpstAPCP:relProj}\else\Cref{d:relProj:relProj}\fi).}
    \ifmpstAPCP \label{alg:mpstAPCP:relProj}
    \else \label{alg:relProj:relProj} \fi
\end{algorithm}

\begin{definition}[Relative Projection]
    \ifmpstAPCP \label{d:mpstAPCP:relProj}
    \else \label{d:relProj:relProj} \fi
    Given a global type $\gtc{G}$, we define its relative projection onto a pair of participants $(p,q)$, denoted $\gtc{G} \wrt (p,q)$, by induction on the structure of $\gtc{G}$ as given in \ifmpstAPCP\Cref{alg:mpstAPCP:relProj}\else\Cref{alg:relProj:relProj}\fi, using the auxiliary function $\ddep$ defined as follows:
    \begin{align*}
        \ddep((p,q), s {!} r \gtBraces{ \msg i<S_i> \gtc. \gtc{G_i} }_{i \in I}) \deq \begin{cases}
            \ifmpstAPCP
            \tSkip \rtc. (\gtc{G_{i'}} \wrt (p,q)) ~\text{[$i' \in I$]}
            \else
            \gtc{G_{i'}} \wrt (p,q) ~\text{[$i' \in I$]}
            \fi
            & \text{if $\forall i{,}j. \gtc{G_i} \wrt (p,q) = \gtc{G_j} \wrt (p,q)$}
            \\
            (p {!} r) {!} q \rtBraces{ i \rtc. (\gtc{G_i} \wrt (p,q)) }_{i \in I}
            & \text{if $p=s$}
            \\
            (q {!} r) {!} p \rtBraces{ i \rtc. (\gtc{G_i} \wrt (p,q)) }_{i \in I}
            & \text{if $q=s$}
            \\
            (p {?} s) {!} q \rtBraces{ i \rtc. (\gtc{G_i} \wrt (p,q)) }_{i \in I}
            & \text{if $p=r$}
            \\
            (q {?} s) {!} p \rtBraces{ i \rtc. (\gtc{G_i} \wrt (p,q)) }_{i \in I}
            & \text{if $q=r$}
        \end{cases}
    \end{align*}
\end{definition}

We discuss how \ifmpstAPCP\Cref{d:mpstAPCP:relProj} \else\Cref{d:relProj:relProj} \fi projects global types onto a pair of participants $(p,q)$, as per \ifmpstAPCP\Cref{alg:mpstAPCP:relProj}\else\Cref{alg:relProj:relProj}\fi.
The most interesting case is the projection of a direct exchange $s {!} r \gtBraces{ \msg i<S_i> \gtc. \gtc{G_i} }_{i \in I}$ (\ifmpstAPCP\cref{li:mpstAPCP:rpEx}\else\cref{li:relProj:rpEx}\fi).
When the exchange involves both $p$ and $q$, the projection yields an exchange between $p$ and $q$ with the appropriate sender.
Otherwise, the projection relies on the function `$\,\ddep(\ldots)\,$', which determines whether the exchange is a non-local choice for~$p$ and $q$ and yields an appropriate projection accordingly:
\begin{itemize}
    \item
        If the projections of all branches are equal, the exchange is not a non-local choice and $\ddep$ yields \ifmpstAPCP the unobservable step $\tSkip$ followed by \fi the projection of any branch.
    \item
        If there are branches with different projections, the exchange is a non-local choice, so $\ddep$ yields a dependency if possible.
        If $p$ or $q$ is involved, $\ddep$ yields an appropriate dependency (e.g., $p {!} r$ if $p$ is the sender, or $q {?} s$ if $q$ is the recipient).
        If neither $p$ nor $q$ are involved, then $\ddep$ cannot yield a dependency and projection is thus undefined.
\end{itemize}

The projection of $\mu X \gtc. \gtc{G'}$ (\ifmpstAPCP\cref{li:mpstAPCP:rpRec}\else\cref{li:relProj:rpRec}\fi) considers the projection of the body $\gtc{G'} \wrt (p,q)$ to see whether $p$ and $q$ interact in $\gtc{G'}$.
If $\gtc{G'} \wrt (p,q)$ is not contractive on $X$ (\ifmpstAPCP\Cref{d:mpstAPCP:relProj:contr}\else\Cref{d:relProj:contr}\fi), it is \ifmpstAPCP a (possibly empty) sequence of $\tSkip$s followed by \fi $\tEnd$ or $X$: $p$ and $q$ do not interact and the projection yields $\tEnd$ (\ifmpstAPCP\cref{li:mpstAPCP:rpRecEnd}\else\cref{li:relProj:rpRecEnd}\fi).
Otherwise, $p$ and~$q$ do interact and projection preserves the recursive definition (\ifmpstAPCP\cref{li:mpstAPCP:rpRecRec}\else\cref{li:relProj:rpRecRec}\fi).
The projections of recursive calls $X$ (\ifmpstAPCP\cref{li:mpstAPCP:rpCall}\else\cref{li:relProj:rpCall}\fi)\ifmpstAPCP, $\tEnd$ (\cref{li:mpstAPCP:rpEnd}), and $\tSkip$ (\cref{li:mpstAPCP:rpSkip}) \else~and $\tEnd$ (\cref{li:relProj:rpEnd}) \fi are homomorphic.

\begin{example}[Projections of $\gtc{G_{\sff{auth}}}$]
    \ifmpstAPCP \label{x:mpstAPCP:GauthProj}
    \else \label{x:relProj:GauthProj} \fi
    To illustrate relative projection, let us reconsider~$\gtc{G_{\sff{auth}}}$:
    \[
        \gtc{G_{\sff{auth}}} = \mu X \gtc. s {!} c \gtBraces*{
            \begin{array}{@{}l@{}}
                \sff{login} \gtc. c {!} a (\msg \sff{passwd} <\sff{str}>) \gtc. a {!} s : \msg \sff{auth} <\sff{bool}> \gtc. X \gtc,
                \\
                \sff{quit} \gtc. c {!} a (\sff{quit}) \gtc. \tEnd
            \end{array}
        }
    \]
    The relative projection onto $(s,c)$ is straightforward, as there are no non-local choices:
    \[
        \gtc{G_{\sff{auth}}} \wrt (s,c) = \mu X \rtc. s {!} c \rtBraces*{
            \begin{array}{@{}l@{}}
                \sff{login} \rtc. \ifmpstAPCP\tSkip^2 \rtc.\fi X \rtc,
                \\
                \sff{quit} \rtc. \ifmpstAPCP\tSkip \rtc.\fi \tEnd
            \end{array}
        }
    \]
    However, compare the projection of the initial \sff{login} branch onto $(s,a)$ and $(c,a)$ with the projection of the \sff{quit} branch: they are different.
    Therefore, the initial exchange between $s$ and $c$ is a non-local choice in the protocols relative to $(s,a)$ and $(c,a)$.
    Since $s$ is involved in this exchange, the non-local choice is detected by $\ddep$:
    \[
        \ddep((s,a), s {!} c \gtBraces{ \sff{login} \ldots \gtc, ~~ \sff{quit} \ldots }) = (s {!} c) {!} a \rtBraces{ \sff{login} \ldots \rtc, ~~ \sff{quit} \ldots }
    \]
    Hence, this non-local choice is included in the relative projection onto $(s,a)$ as a dependency:
    \[
        \gtc{G_{\sff{auth}}} \wrt (s,a) = \mu X \rtc. (s {!} c) {!} a \rtBraces*{
            \begin{array}{@{}l@{}}
                \sff{login} \rtc. \ifmpstAPCP\tSkip \rtc.\fi a {!} s (\msg \sff{auth}<\sff{bool}>) \rtc. X \rtc,
                \\
                \sff{quit} \rtc. \ifmpstAPCP\tSkip \rtc.\fi \tEnd
            \end{array}
        }
    \]
    Similarly, $c$ is involved in the initial exchange, so the non-local choice can also be included in the relative projection onto $(c,a)$ as a dependency:
    \[
        \gtc{G_{\sff{auth}}} \wrt (c,a) = \mu X \rtc. (c {?} s) {!} a \rtBraces*{
            \begin{array}{@{}l@{}}
                \sff{login} \rtc. c {!} a (\msg \sff{passwd}<\sff{str}>) \rtc. \ifmpstAPCP\tSkip \rtc.\fi X \rtc,
                \\
                \sff{quit} \rtc. c {!} a (\msg \sff{quit}<\sff{unit}>) \rtc. \tEnd
            \end{array}
        }
        \tag*{\qedhere}
    \]
\end{example}

The input order of participants for projection does not matter:

\begin{proposition}
    \ifmpstAPCP \label{p:mpstAPCP:relProj:sensible}
    \else \label{p:relProj:sensible} \fi
    Suppose given a global type $\gtc{G}$ and distinct participants $p,q \in \prt(\gtc{G})$.
    \begin{itemize}
        \item
            If $\gtc{G} \wrt (p,q)$ is defined, then $\gtc{G} \wrt (p,q) = \gtc{G} \wrt (q,p)$ and $\prt(\gtc{G} \wrt (p,q)) \subseteq \{p,q\}$;
        \item
            $\gtc{G} \wrt (p,q)$ is undefined if and only if $\gtc{G} \wrt (q,p)$ is undefined.
    \end{itemize}
\end{proposition}

\ifmpstAPCP\else
\paragraph{Well-formed Global Types.}
\fi

We now define \emph{well-formedness} for global types.
\ifmpstAPCP\else
Unlike usual \MPST approaches, our definition relies exclusively on (relative) projection (\Cref{d:relProj:relProj}), and does not appeal to external notions such as merge and subtyping~\cite{journal/lmcs/HuBYD12,conf/icdcit/YoshidaG20}.
\fi

\begin{definition}[Relative Well-Formedness]
    \ifmpstAPCP \label{d:mpstAPCP:rwf}
    \else \label{d:relProj:rwf} \fi
    A global type $\gtc{G}$ is \emph{relative well-formed} if, for every distinct $p,q \in \prt(\gtc{G})$, the projection $\gtc{G} \wrt (p,q)$ is defined.
\end{definition}

\ifthesis\else
The following contrasts our new notion of relative well-for\-med\-ness with notions of well-formedness based on the usual notion of local types~\cite{conf/popl/HondaYC08,conf/icalp/DenielouY13}.

\begin{example}
    \label{x:relProj:wfs}
    Consider the following global type involving participants $p, q, r, s$:
    \[
        \gtc{G_{\sff{m}}} \deq p {!} q \gtBraces*{
            \begin{array}{@{}l@{}}
                \msg 1<S_a> \gtc. p {!} r (\msg 1<S_b>) \gtc. p {!} s (\msg 1<S_c>) \gtc. q {!} r (\msg 1<S_d>) \gtc. q {!} s (\msg 1<S_e>) \gtc. \tEnd \gtc,
                \\
                \msg 2<S_f> \gtc. r {!} p (\msg 2<S_g>) \gtc. s {!} p (\msg 2<S_h>) \gtc. r {!} q (\msg 2<S_i>) \gtc. s {!} q (\msg 2<S_j>) \gtc. \tEnd
            \end{array}
        }
    \]
    The initial exchange between $p$ and $q$ is a non-local choice influencing the protocols between other pairs of participants.
    Well-formedness as in~\cite{conf/popl/HondaYC08,conf/icalp/DenielouY13} forbids non-local choices.
    In contrast, $\gtc{G_{\sff{m}}}$ is relative well-formed: $p$ and $q$ must both forward the selected label to both $r$ and $s$.
    The dependencies in the following relative projections express precisely this:
    \begin{align*}
        \gtc{G_{\sff{m}}} \wrt (p,r)
        &= (p {!} q) {!} r \rtBraces{ 1 \rtc. p {!} r (\msg 1<S_b>) \rtc. \ifmpstAPCP\tSkip^3 \rtc.\fi \tEnd \rtc, ~~ 2 \rtc. r {!} p (\msg 2<S_g>) \rtc. \ifmpstAPCP\tSkip^3 \rtc.\fi \tEnd }
        \\
        \gtc{G_{\sff{m}}} \wrt (p,s)
        &= (p {!} q) {!} s \rtBraces{ 1 \rtc. \ifmpstAPCP\tSkip \rtc.\fi p {!} s (\msg 1<S_c>) \rtc. \ifmpstAPCP\tSkip^2 \rtc.\fi \tEnd \rtc, ~~ 2 \rtc. \ifmpstAPCP\tSkip \rtc.\fi s {!} p (\msg 2<S_h>) \rtc. \ifmpstAPCP\tSkip^2 \rtc.\fi \tEnd }
        \\
        \gtc{G_{\sff{m}}} \wrt (q,r)
        &= (q {?} p) {!} r \rtBraces{ 1 \rtc. \ifmpstAPCP\tSkip^2 \rtc.\fi q {!} r (\msg 1<S_d>) \rtc. \ifmpstAPCP\tSkip \rtc.\fi \tEnd \rtc, ~~ 2 \rtc. \ifmpstAPCP\tSkip^2 \rtc.\fi r {!} q (\msg 2<S_i>) \rtc. \ifmpstAPCP\tSkip \rtc.\fi \tEnd }
        \\
        \gtc{G_{\sff{m}}} \wrt (q,s)
        &= (q {?} p) {!} s \rtBraces{ 1 \rtc. \ifmpstAPCP\tSkip^3 \rtc.\fi q {!} s (\msg 1<S_e>) \rtc. \tEnd \rtc, ~~ 2 \rtc. \ifmpstAPCP\tSkip^3 \rtc.\fi s {!} q (\msg 2<S_j>) \rtc. \tEnd }
        \tag*{\qedhere}
    \end{align*}
\end{example}
\fi

\ifmpstAPCP\else
Dependencies in relative types follow the non-local choices in the given global type: by implementing such choices, dependencies ensure correct projectability.
They induce additional messages, but in our view this is an acceptable price to pay for an expressive notion of well-formedness based only on projection.
It is easy to see that in a global type with $n$ participants, the number of messages per communication is $\mcl{O}(n)$---an upper-bound following from the worst-case scenario in which both sender and recipient have to forward a label to $n-2$ participants due to dependencies, as in the example above.
However, in prac\-tice, sender and recipient will rarely both have to forward labels, let alone both to all participants.
\fi

\section{Analyzing Global Types using Routers}
\label{s:routers}

\begin{figure}[t]
    \ifmpstAPCP\else \begin{mdframed} \fi
        \mbox{}\hfill%
        \begin{minipage}[t]{.51\textwidth}
            \vspace{1em}
            \hfill%
            \ifmpstAPCP\begin{tikzpicture}
    \begin{scope}[local bounding box=bndP]
        \node (P) [] at (0, 0) {$P$};
        \node (rtrP) [right=4mm of P] {$\lmed{c}{\{s,a\}}{\gtc{G_{\sff{auth}}}}$};
    \end{scope}
    \node [fit=(bndP), shape=rectangle, draw=black, inner xsep=0mm, inner ysep=1mm] (boxP) {};
    \node [below=0mm of boxP] {${} \in \lsys(\gtc{G_{\sff{auth}}}, \{c\})$};
    \draw[-, cblPink] (P) -- (rtrP);

    \begin{scope}[local bounding box=bndQ]
        \node (rtrQ) [right=1cm of rtrP, yshift=12mm] {$\lmed{s}{\{a,s\}}{\gtc{G_{\sff{auth}}}}$};
        \node (Q) [right=4mm of rtrQ] {$Q$};
    \end{scope}
    \node [fit=(bndQ), shape=rectangle, draw=black, inner xsep=1mm, inner ysep=0mm] (boxQ) {};
    \node [below=0mm of boxQ] {${} \in \lsys(\gtc{G_{\sff{auth}}}, \{s\})$};
    \draw[-, cblPink] (Q) -- (rtrQ);

    \begin{scope}[local bounding box=bndR]
        \node (rtrR) [right=1cm of rtrP, yshift=-12mm] {$\lmed{a}{\{s,c\}}{\gtc{G_{\sff{auth}}}}$};
        \node (R) [right=4mm of rtrR] {$R$};
    \end{scope}
    \node [fit=(bndR), shape=rectangle, draw=black, inner xsep=1mm, inner ysep=0mm] (boxR) {};
    \node [below=0mm of boxR] {${} \in \lsys(\gtc{G_{\sff{auth}}}, \{a\})$};
    \draw[-, cblPink] (R) -- (rtrR);

    \draw[-, cblPurple] (rtrP) -- (rtrQ.west);
    \draw[-, cblPurple] (rtrQ.west) to [out=225, in=135] (rtrR.west);
    \draw[-, cblPurple] (rtrR.west) -- (rtrP);
\end{tikzpicture}%
            \else\begin{tikzpicture}
    \begin{scope}[local bounding box=bndP]
        \node (P) [] at (0, 0) {$P$};
        \node (rtrP) [right=4mm of P] {$\lmed{c}{\{s,a\}}{\gtc{G_{\sff{auth}}}}$};
    \end{scope}
    \node [fit=(bndP), shape=rectangle, draw=black, inner xsep=0mm, inner ysep=1mm] (boxP) {};
    \node [below=0mm of boxP] {${} \in \lsys(\gtc{G_{\sff{auth}}}, \{c\})$};
    \draw[-, cblPink] (P) -- (rtrP);

    \begin{scope}[local bounding box=bndQ]
        \node (rtrQ) [right=1cm of rtrP, yshift=12mm] {$\lmed{s}{\{a,s\}}{\gtc{G_{\sff{auth}}}}$};
        \node (Q) [right=4mm of rtrQ] {$Q$};
    \end{scope}
    \node [fit=(bndQ), shape=rectangle, draw=black, inner xsep=1mm, inner ysep=0mm] (boxQ) {};
    \node [below=0mm of boxQ] {${} \in \lsys(\gtc{G_{\sff{auth}}}, \{s\})$};
    \draw[-, cblPink] (Q) -- (rtrQ);

    \begin{scope}[local bounding box=bndR]
        \node (rtrR) [right=1cm of rtrP, yshift=-12mm] {$\lmed{a}{\{s,c\}}{\gtc{G_{\sff{auth}}}}$};
        \node (R) [right=4mm of rtrR] {$R$};
    \end{scope}
    \node [fit=(bndR), shape=rectangle, draw=black, inner xsep=1mm, inner ysep=0mm] (boxR) {};
    \node [below=0mm of boxR] {${} \in \lsys(\gtc{G_{\sff{auth}}}, \{a\})$};
    \draw[-, cblPink] (R) -- (rtrR);

    \draw[-, cblPurple] (rtrP) -- (rtrQ.west);
    \draw[-, cblPurple] (rtrQ.west) to [out=225, in=135] (rtrR.west);
    \draw[-, cblPurple] (rtrR.west) -- (rtrP);
\end{tikzpicture}\fi%
            \hfill
            \vspace{-1em}
        \end{minipage}%
        \hfill\vline\hfill%
        \begin{minipage}[t]{.45\textwidth}
            \vspace{1em}
            \hfill%
            \ifmpstAPCP\begin{tikzpicture}
    \begin{scope}[local bounding box=bndP]
        \node (P) [] at (0, 0) {$P$};
        \node (rtrP) [right=4mm of P] {$\lmed{c}{\{s,a\}}{\gtc{G_{\sff{auth}}}}$};
    \end{scope}
    \node [fit=(bndP), shape=rectangle, draw=black, inner xsep=0mm, inner ysep=1mm] (boxP) {};
    \node [below=0mm of boxP] {${} \in \lsys(\gtc{G_{\sff{auth}}}, \{c\})$};
    \draw[-, cblPink] (P) -- (rtrP);

    \begin{scope}[local bounding box=bndQ]
        \node (rtrQ) [right=10mm of rtrP, yshift=12mm] {$\lmed{s}{\{a,s\}}{\gtc{G_{\sff{auth}}}}$};
        \node (rtrR) [right=10mm of rtrP, yshift=-12mm] {$\lmed{a}{\{s,c\}}{\gtc{G_{\sff{auth}}}}$};
        \node (S) [right=17.5mm of rtrP] {$S$};
    \end{scope}
    \node [fit=(bndQ), shape=rectangle, draw=black, inner xsep=1mm, inner ysep=0mm] (boxQ) {};
    \node [below=0mm of boxQ] {${} \in \lsys(\gtc{G_{\sff{auth}}}, \{s,a\})$};
    \draw[-, cblPink] (S) -- (rtrQ);
    \draw[-, cblPink] (S) -- (rtrR);

    \draw[-, cblPurple] (rtrP) -- (rtrQ.west);
    \draw[-, cblPurple] (rtrP) -- (rtrR.west);

    \draw[-, cblPurple] (rtrQ.west) to [out=225, in=135] (rtrR.west);
\end{tikzpicture}%
            \else\begin{tikzpicture}
    \begin{scope}[local bounding box=bndP]
        \node (P) [] at (0, 0) {$P$};
        \node (rtrP) [right=4mm of P] {$\lmed{c}{\{s,a\}}{\gtc{G_{\sff{auth}}}}$};
    \end{scope}
    \node [fit=(bndP), shape=rectangle, draw=black, inner xsep=0mm, inner ysep=1mm] (boxP) {};
    \node [below=0mm of boxP] {${} \in \lsys(\gtc{G_{\sff{auth}}}, \{c\})$};
    \draw[-, cblPink] (P) -- (rtrP);

    \begin{scope}[local bounding box=bndQ]
        \node (rtrQ) [right=10mm of rtrP, yshift=12mm] {$\lmed{s}{\{a,s\}}{\gtc{G_{\sff{auth}}}}$};
        \node (rtrR) [right=10mm of rtrP, yshift=-12mm] {$\lmed{a}{\{s,c\}}{\gtc{G_{\sff{auth}}}}$};
        \node (S) [right=17.5mm of rtrP] {$S$};
    \end{scope}
    \node [fit=(bndQ), shape=rectangle, draw=black, inner xsep=1mm, inner ysep=0mm] (boxQ) {};
    \node [below=0mm of boxQ] {${} \in \lsys(\gtc{G_{\sff{auth}}}, \{s,a\})$};
    \draw[-, cblPink] (S) -- (rtrQ);
    \draw[-, cblPink] (S) -- (rtrR);

    \draw[-, cblPurple] (rtrP) -- (rtrQ.west);
    \draw[-, cblPurple] (rtrP) -- (rtrR.west);

    \draw[-, cblPurple] (rtrQ.west) to [out=225, in=135] (rtrR.west);
\end{tikzpicture}\fi%
            \hfill
            \vspace{-1em}
        \end{minipage}%
        \hfill\mbox{}
    \ifmpstAPCP\else \end{mdframed} \fi

    \caption{
        Two different networks of routed implementations for $\gtc{G_{\sff{auth}}}$~\eqref{eq:mpstAPCP:Gauth}, without interleaving (left) and with interleaving (right).
        For participants $p$ and $\tilde{q}$, \Cref{d:mpstAPCP:router} gives the router process $\lmed{p}{\tilde{q}}{G}$ and \Cref{d:mpstAPCP:routedImplementations} gives the set $\lsys(G, \tilde{q})$.
        Lines indicate channels and boxes are local compositions of processes.
    }
    \label{f:mpstAPCP:GauthNetwork}
\end{figure}

In this section, we develop our decentralized analysis of multiparty protocols (\Cref{s:mpstAPCP:mpst}) using relative types and \APCP (\ifappendix\Cref{s:mpstAPCP:APCP}\else{}cf.\ \Cref{c:APCP} on \Cpageref{c:APCP}\fi).
The intended setup is as follows.
Each participant's role in a global type $\gtc{G}$ is implemented by a process, which is connected to a \emph{router}: a process that orchestrates the participant's interactions in $\gtc{G}$.
The resulting \emph{routed implementations} (\Cref{d:mpstAPCP:routedImplementations}) can then directly connect to each other to form a decentralized \emph{network of routed implementations} that implements $\gtc{G}$.
This way we realize the scenario sketched in \Cref{f:mpstAPCP:chorOrch}~(left), which is featured in more detail in \Cref{f:mpstAPCP:GauthNetwork}~(left).

Key in our analysis is the \emph{synthesis} of a participant's router from a global type (\Cref{s:mpstAPCP:analysis:routers}).
To assert well-typedness---and thus deadlock-freedom---of networks of routed implementations (\Cref{t:mpstAPCP:routerTypes}), we extract binary session types from the global type and its associated relative types (\Cref{s:mpstAPCP:analysis:typing}):
\begin{itemize}
    \item
        from global types we extract types for channels between implementations and routers;
    \item
        from relative types we extract types for channels between pairs of routers.
\end{itemize}

After defining routers and showing their typability, we set up networks of routed implementations of global types (\Cref{s:mpstAPCP:analysis:networks}).
To enable the transference of deadlock-freedom from \APCP to multiparty protocols, we then establish an operational correspondence between global types and networks of routed implementations (\Cref{t:mpstAPCP:completeness,t:mpstAPCP:soundness}).
Finally, to show that our routed approach strictly generalizes the prior centralized analyses~\cite{conf/forte/CairesP16,conf/concur/CarboneLMSW16}, we define an orchestrated analysis of global types and show that it is behaviorally equivalent to a centralized composition of routers (\Cref{s:mpstAPCP:analysis:mediums}).

In the following section (\Cref{s:mpstAPCP:examples}), we wil show routers in action.

\subsection{Synthesis of Routers}
\label{s:mpstAPCP:analysis:routers}

We synthesize routers by decomposing each exchange in the global type into four sub-steps, which we motivate by considering the initial exchange from $s$ to $c$ in $\gtc{G_{\sff{auth}}}$~\eqref{eq:mpstAPCP:Gauth}: $s {!} c \gtBraces{ \sff{login} \gtc\ldots \gtc, \sff{quit} \gtc\ldots }$.
As explained in \Cref{x:mpstAPCP:GauthProj}, this exchange induces a dependency in the relative projections of $\gtc{G_{\sff{auth}}}$ onto $(s,a)$ and $(c,a)$.
We decompose this initial exchange as follows, where $P$, $Q$, and $R$ are the implementations of $c$, $s$, and $a$, respectively (given in \Cref{x:mpstAPCP:impl}) and  $\rtr{x}$ stands for the router of each $x \in \{s,c,a\}$.
Below, multiple actions in one step happen concurrently:
\begin{enumerate}
    \item
        $Q$ sends $\ell \in \{\sff{login},\sff{quit}\}$ to $\rtr{s}$.
    \item
        $\rtr{s}$ sends $\ell$ to $\rtr{c}$ (recipient) and $\rtr{a}$ (output dependency).
        $Q$ sends \sff{unit} value $v$ to~$\rtr{s}$.
    \item
        $\rtr{c}$ sends $\ell$ to $P$ and $\rtr{a}$ (input dependency).
        $\rtr{s}$ forwards $v$ to $\rtr{c}$.
    \item
        $\rtr{c}$ forwards $v$ to $P$. $\rtr{a}$ sends $\ell$ to $R$.
\end{enumerate}
In \Cref{s:mpstAPCP:analysis:typing}, we follow this decomposition to assign to each consecutive step a consecutive priority: this ensures the consistency of priority checks required to establish the deadlock-freedom of networks of routed implementations.

We define an algorithm that synthesizes a \emph{router process} for a given global type and participant.
More precisely: given $G$, a participant $p$, and $\tilde{q} = \prt(\gtc{G}) \setminus \{p\}$, the algorithm generates a process, denoted $\lmed{p}{\tilde{q}}{\gtc{G}}$, which connects with a process implementing $p$'s role in $G$ on channel $\ci{\mu}{p}$; we shall write such channels in \textcolor{ichClr}{\ichClrName}.
This router for $p$ connects with the routers of the other participants in $\gtc{G}$ ($q_i \in \tilde{q}$) on channels $\crt{p}{q_1}, \ldots, \crt{p}{q_n}$; we shall write such channels in \textcolor{rchClr}{\rchClrName}.
(This convention explains the colors of the lines in \Cref{f:mpstAPCP:GauthNetwork}.)

The router synthesis algorithm relies on relative projection to detect non-local choices; this way, the router can synchronize with the participant's implementation and with other routers appropriately.
To this end, we define the predicate `$\hdep$', which is true for an exchange and a pair of participants if the exchange induces a dependency for either participant.
Recall that relative projection produces a $\tSkip$ when an exchange is not non-local (cf.\ \Cref{d:mpstAPCP:relProj}).
Thus, $\hdep$ only holds true if relative projection does not produce a $\tSkip$.

\begin{definition}
    \label{d:mpstAPCP:hdep}
        The predicate $\hdep(q,p,\gtc{G})$ is true if and only if
        \begin{itemize}
            \item
                $\gtc{G} = s {!} r \gtBraces{ \msg i<S_i> \gtc. \gtc{G_i} }_{i \in I}$ and $q \notin \{s,r\}$ and $p \in \{s,r\}$, and
            \item
                $\ddep((p,q),\gtc{G}) \neq \tSkip \rtc. \rtc{R}$ for all relative types $\rtc{R}$, where $\ddep$ is as in \Cref{d:mpstAPCP:relProj}.
        \end{itemize}
\end{definition}

\begin{example}
    Consider the type $\gtc{G_{\sff{h}}} \deq p {!} q \gtBraces{ \sff{a} \gtc. p {!} r (\sff{a}) \gtc. \tEnd \gtc, \sff{b} \gtc. r {!} p (\sff{b}) \gtc. \tEnd }$.
    Then $\hdep(q, p, \gtc{G_{\sff{h}}})$ is false because the initial exchange in $\gtc{G_{\sff{h}}}$ is not a dependency for $p$ and $q$, but $\hdep(r, p, \gtc{G_{\sff{h}}})$ is true because the initial exchange in $\gtc{G_{\sff{h}}}$ is indeed a dependency for $p$ and~$r$.
\end{example}

\begin{algorithm}[t]
    \DontPrintSemicolon
    \SetAlgoNoEnd
    \SetInd{.2em}{.7em}
    \SetSideCommentRight
    \Def{$\lmed{p}{\tilde{q}}{\gtc{G}}$}{
        \Switch{$\gtc{G}$}{
            \uCase{$s {!} r \gtBraces{ \msg i<S_i> \gtc. \gtc{G_i} }_{i \in I}$}{ \label{li:mpstAPCP:rtrComm}
                $\deps \deq \{ q \in \tilde{q} \mid \hdep(q,p,\gtc{G}) \}$ \; \label{li:mpstAPCP:rtrDeps}

                \lIf{$p = s$}{ \label{li:mpstAPCP:rtrSend}
                    \KwRet
                    $\pBra* \ci{\mu}{p} > {\big\{ i : \pSel* \crt{p}{r} < i \cdot {(\pSel* \crt{p}{q} < i)}_{q \in \deps}
                    \cdot \pIn \ci{\mu}{p}(v) ; \pOut* \crt{p}{r}[w] \cdot (\pFwd [v<>w] \| \lmed{p}{\tilde{q}}{\gtc{G_i}}) \big\}_{i \in I}}$
                }

                \lElseIf{$p = r$}{ \label{li:mpstAPCP:rtrRecv}
                    \KwRet
                    $\pBra* \crt{p}{s} > {\big\{ i : \pSel* \ci{\mu}{p} < i \cdot {(\pSel* \crt{p}{q} < i)}_{q \in \deps}
                    \cdot \pIn \crt{p}{s}(v) ; \pOut* \ci{\mu}{p}[w] \cdot (\pFwd [v<>w] \| \lmed{p}{\tilde{q}}{\gtc{G_i}}) \big\}_{i \in I}}$
                }

                \uElseIf{$p \notin \{s,r\}$}{ \label{li:mpstAPCP:rtrElse}
                    $\dep_s \deq (s \in \tilde{q} \wedge \hdep(p,s,\gtc{G}))$ \; \label{li:mpstAPCP:rtrDepS}
                    $\dep_r \deq (r \in \tilde{q} \wedge \hdep(p,r,\gtc{G}))$ \; \label{li:mpstAPCP:rtrDepR}

                    \lIf{$\dep_s$ and $\neg \dep_r$}{ \label{li:mpstAPCP:rtrOnlyS}
                        \KwRet
                        $\pBra* \crt{p}{s} > {\big\{ i : \pSel* \ci{\mu}{p} < i \cdot \lmed{p}{\tilde{q}}{\gtc{G_i}} \big\}_{i \in I}}$
                    }

                    \lElseIf{$\dep_r$ and $\neg \dep_s$}{ \label{li:mpstAPCP:rtrOnlyR}
                        \KwRet
                        $\pBra* \crt{p}{r} > {\big\{i: \pSel* \ci{\mu}{p} < i \cdot \lmed{p}{\tilde{q}}{\gtc{G_i}} \big\}_{i \in I}}$
                    }

                    \lElseIf{$\dep_s$ and $\dep_r$}{ \label{li:mpstAPCP:rtrBoth}
                        \KwRet
                        $\pBra* \crt{p}{s} > {\big\{ i : \pSel* \ci{\mu}{p} < i \cdot \pBra* \crt{p}{r} > {\{ i : \lmed{p}{\tilde{q}}{\gtc{G_i}}\}} \big\}_{i \in I}}$
                    }

                    \lElse{ \label{li:mpstAPCP:rtrElseElse}
                        \KwRet $\lmed{p}{\tilde{q}}{\gtc{G_j}}$ for any $j \in I$
                    }
                }
            }

            \uCase{$\mu X \gtc. \gtc{G'}$}{ \label{li:mpstAPCP:rtrRecDef}
                $\tilde{q}' \deq \{ q \in \tilde{q} \mid \gtc{G} \wrt (p,q) \neq \tEnd \}$ \; \label{li:mpstAPCP:rtrRecComp}
                \lIf{$\tilde{q}' \neq \emptyset$}{ \label{li:mpstAPCP:rtrRecRet}
                    \KwRet $\pRec X(\ci{\mu}{p}, {(\crt{p}{q})}_{q \in \tilde{q}'}) ; \lmed{p}{\tilde{q}'}{\gtc{G'}}$
                }
                \lElse{ \label{li:mpstAPCP:rtrRecInact}
                    \KwRet $\0$
                }
            }

            \lCase{$X$}{ \label{li:mpstAPCP:rtrRecCall}
                \KwRet $\pCall X<\ci{\mu}{p}, {(\crt{p}{q})}_{q \in \tilde{q}}>$
            }

            \lCase{$\tSkip \gtc. \gtc{G'}$}{ \label{li:mpstAPCP:rtrSkip}
                \KwRet $\lmed{p}{\tilde{q}}{\gtc{G'}}$
            }

            \lCase{$\tEnd$}{\label{li:mpstAPCP:rtrEnd}
                \KwRet $\0$
            }
        }
    }

    \caption{Synthesis of Router Processes (\Cref{d:mpstAPCP:router}).}
    \label{alg:mpstAPCP:router}
\end{algorithm}

\begin{definition}[Router Synthesis]
    \label{d:mpstAPCP:router}
    Given a global type $\gtc{G}$, a participant $p$, and participants~$\tilde{q}$, \Cref{alg:mpstAPCP:router} defines the synthesis of a \emph{router process}, denoted $\lmed{p}{\tilde{q}}{\gtc{G}}$, that interfaces the interactions of $p$ with the other protocol participants according to $G$.
\end{definition}

\noindent
We often write $\rtr{p}$ for $\lmed{p}{\prt(\gtc{G}) \setminus \{p\}}{\gtc{G}}$ when $\gtc{G}$ is clear from the context.

\Cref{alg:mpstAPCP:router} distinguishes six cases depending on the syntax of $\gtc{G}$ (\Cref{d:mpstAPCP:globalTypes}).
The key case is ${s {!} r \gtBraces{ \msg i<U_i> \gtc. \gtc{G_i} }_{i \in I}}$ (\cref{li:mpstAPCP:rtrComm}).
First, the algorithm computes a set $\deps$ of participants that depend on the exchange using $\hdep$ (cf.\ \Cref{d:mpstAPCP:hdep}).
Then, the algorithm considers the three possibilities for $p$:
\begin{enumerate}
    \item
        If $p = s$ then $p$ is the sender (\cref{li:mpstAPCP:rtrSend}): the algorithm returns a process that receives a label~\mbox{$i \in I$} over $\ci{\mu}{p}$; sends $i$ over $\crt{p}{r}$ and over $\crt{p}{q}$ for every $q \in \deps$; receives a channel $v$ over~$\ci{\mu}{p}$; forwards $v$ as $w$ over $\crt{p}{r}$; and continues as $\lmed{p}{\tilde{q}}{\gtc{G_i}}$.

    \item
        If $p = r$ then $p$ is the recipient (\cref{li:mpstAPCP:rtrRecv}): the algorithm returns a process that receives a label $i \in I$ over $\crt{p}{s}$; sends $i$ over $\ci{\mu}{p}$ and over $\crt{p}{q}$ for every $q \in \deps$; receives a channel $v$ over $\crt{p}{s}$; forwards $v$ as $w$ over $\ci{\mu}{p}$; and continues as $\lmed{p}{\tilde{q}}{\gtc{G_i}}$.

    \item\label{itm:rtrDeps}
        Otherwise, if $p$ is not involved (\cref{li:mpstAPCP:rtrElse}), we use $\hdep$ to determine whether $p$ depends on an output from $s$, an input from $r$, or on both (\cref{li:mpstAPCP:rtrDepS,li:mpstAPCP:rtrDepR}).
        If $p$ only depends on the output from $s$, the algorithm returns a process that receives a label $i \in I$ over~$\crt{p}{s}$; sends $i$ over $\ci{\mu}{p}$; and continues as $\lmed{p}{\tilde{q}}{\gtc{G_i}}$ (\cref{li:mpstAPCP:rtrOnlyS}).
        If $p$ only depends on an input from $r$, the returned process is similar; the only difference is that $i$ is received over~$\crt{p}{r}$ (\cref{li:mpstAPCP:rtrOnlyR}).

        When $p$ depends on \emph{both} the output from $s$ and on the input from $r$ (\cref{li:mpstAPCP:rtrBoth}), the algorithm returns a process that receives a label $i \in I$ over $\crt{p}{s}$; sends $i$ over $\ci{\mu}{p}$; receives the label $i$ over $\crt{p}{r}$; and continues as $\lmed{p}{\tilde{q}}{\gtc{G_i}}$.

        If there are no dependencies, the returned process is $\lmed{p}{\tilde{q}}{\gtc{G_j}}$, for arbitrary $j \in I$ (\cref{li:mpstAPCP:rtrElseElse}).
\end{enumerate}

\noindent
In case $\mu X \gtc. \gtc{G'}$ (\cref{li:mpstAPCP:rtrRecDef}), the algorithm stores in $\tilde{q}'$ those $q \in \tilde{q}$ that interact with $p$ in $\gtc{G'}$ (i.e.\ $(\mu X \gtc. \gtc{G'}) \wrt (p,q) \neq \tEnd$).
Then, if $\tilde{q}'$ is non-empty (\cref{li:mpstAPCP:rtrRecRet}), the algorithm returns a recursive definition with as context the channels $\crt{p}{q}$ for $q \in \tilde{q}'$ and $\ci{\mu}{p}$.
Otherwise, the algorithm returns~$\0$~(\cref{li:mpstAPCP:rtrRecInact}).
In case $X$ (\cref{li:mpstAPCP:rtrRecCall}), the algorithm returns a recursive call with as context the channels $\crt{p}{q}$ for $q \in \tilde{q}$ and $\ci{\mu}{p}$.
In case $\tSkip \gtc. \gtc{G'}$ (\cref{li:mpstAPCP:rtrSkip}), it continues with $\gtc{G'}$ immediately.
Finally, in case $\tEnd$ (\cref{li:mpstAPCP:rtrEnd}), the algorithm returns $\0$.

Considering the number of steps required to return a process, the complexity of \Cref{alg:mpstAPCP:router} is linear in the size of the given global type (defined as the sum of the number of communications over all branches).

\subsection{Types for the Router's Channels}
\label{s:mpstAPCP:analysis:typing}

Here, we obtain \ifappendix session types (cf.\ \Cref{d:mpstAPCP:APCP:types}) \else \APCP session types (cf.\ \Cref{c:APCP} on \Cpageref{c:APCP}) \fi for (i)~the channels between routers and implementations (\Cref{s:mpstAPCP:analysis:typing:locProj}) and for (ii)~the channels between pairs of routers (\Cref{s:mpstAPCP:analysis:typing:relProp}).
While the former are extracted from global types, the latter are extracted from relative types.

\subsubsection{The Channels between Routers and Implementations}
\label{s:mpstAPCP:analysis:typing:locProj}

We begin with the session types for the channels between routers and implementations (given in \textcolor{ichClr}{\ichClrName}), extracted directly from the global type.
A participant's implementation performs on this channel precisely those actions that the participant must perform as per the global type.
Hence, we define this extraction as a form of \emph{local projection} onto a \emph{single participant}.
The resulting session type may used as a guidance for specifying a participant implementation, which can then connect to the router's dually typed channel endpoint.

\begin{figure}[t]
    \ifthesis\else \begin{mdframed} \fi
        Below, $\pri \in \mbb{N}$ is arbitrary:
        \begin{align*}
            \msgProp{\tEnd}
            &\deq \bullet
            &
            \msgProp{{!}T . S}
            &\deq \msgProp{T} \tensor^\pri \msgProp{S}
            &
            \msgProp{\oplus \{ i : S_i \}_{i \in I}}
            &\deq \oplus^\pri \{ i : \msgProp{S_i} \}_{i \in I}
            \\
            &
            &
            \msgProp{{?}T . S}
            &\deq \msgProp{T} \parr^\pri \msgProp{S}
            &
            \msgProp{\& \{ i : S_i \}_{i \in I}}
            &\deq \&^\pri \{ i : \msgProp{S_i} \}_{i \in I}
        \end{align*}

        \dashes

        If $\gtc{G} = s {!} r \gtBraces{ \msg i<S_i> \gtc. \gtc{G_i} }_{i \in I}$,
        \[
            \gtc{G} \locProjD^\pri p \deq \begin{cases}
                \oplus^\pri \{ i : \msgProp{S_i} \tensor^{\pri+1} (\gtc{G_i} \locProjD^{\pri+4} p) \}_{i \in I}
                & \text{if $p=s$}
                \\
                \&^{\pri+2} \{ i : \ol{\msgProp{S_i}} \parr^{\pri+3} (\gtc{G_i} \locProjD^{\pri+4} p) \}_{i \in I}
                & \text{if $p=r$}
                \\
                \&^{\pri+2} \{ i : (\gtc{G_i} \locProjD^{\pri+4} p) \}_{i \in I}
                & \text{if $p \notin \{s,r\}$ and $\hdep(p,s,\gtc{G})$}
                \\
                \&^{\pri+3} \{ i : (\gtc{G_i} \locProjD^{\pri+4} p) \}_{i \in I}
                & \text{if $p \notin \{s,r\}$ and $\neg \hdep(p,s,\gtc{G})$}
                \\
                & \text{\hphantom{if} and $\hdep(p,r,\gtc{G})$}
                \\
                \gtc{G_{i'}} \locProjD^{\pri+4} p ~\text{[any $i' \in I$]}
                & \text{otherwise}
            \end{cases}
        \]

        Otherwise,
        \begin{align*}
            \tEnd \locProjD^\pri p
            &\deq \bullet
            &
            (\tSkip \gtc. \gtc{G'}) \locProjD^\pri p
            &\deq \gtc{G'} \locProjD^{\pri+4} p
            &
            X \locProjD^\pri p
            &\deq X
            \\
            (\mu X \gtc. \gtc{G'}) \locProjD^\pri p
            &\deq \begin{cases}
                \mu X . (\gtc{G'} \locProjD^\pri p)
                & \text{if $\gtc{G'} \locProjD^\pri p$ defined and contractive on $X$}
                \\
                \bullet
                & \text{otherwise}
            \end{cases}
            \span\span\span\span
        \end{align*}
    \ifthesis\else \end{mdframed} \fi

    \caption{Extracting session types from message types~(top), and local projection: extracting session types from a global type~(bottom, cf.\ \Cref{d:mpstAPCP:locProj}).}
    \label{f:mpstAPCP:locProj}
\end{figure}

Global types contain message types (\Cref{d:mpstAPCP:globalTypes}), so we must first define how we extract session types from message types.
This is a straightforward definition, which leaves priorities unspecified: they do not matter for the typability of routers, which simply forward mes\-sages.
Note that one must still specify these priorities when type-checking implementations, making sure they concur between sender and recipient.

\begin{definition}[From Message Types to Session Types]
    \label{d:mpstAPCP:msgProp}
    We define the extraction of a session type from message type $S$, denoted $\msgProp{S}$, by induction on the structure of $S$ as in \Cref{f:mpstAPCP:locProj}~(top).
\end{definition}

We now define local projection.
To deal with non-local choices, local projection incorporates dependencies by relying on the dependency detection of relative projection (cf.\ \Cref{d:mpstAPCP:relProj}).
Also similar to relative projection, local projection relies on a notion of contractiveness for session types.

\begin{definition}[Contractive Session Types]
    \label{d:mpstAPCP:contrSes}
    Given a session type $A$ and a recursion variable $X$, we say \emph{$A$ is contractive on $X$} if either of the following holds:
    \begin{itemize}
        \item
            $A$ contains a connective in $\{\tensor,\parr,\oplus,\&\}$, or
        \item
            $A$ is a recursive call on a variable other than $X$.
    \end{itemize}
\end{definition}

\begin{definition}[Local Projection: From Global Types to Session Types]
    \label{d:mpstAPCP:locProj}
    We define the local projection of global type $\gtc{G}$ onto participant $p$ with priority $\pri$, denoted $\gtc{G} \locProjD^\pri p$, by induction on the structure of $\gtc{G}$ as in \Cref{f:mpstAPCP:locProj}~(bottom), relying on message type extraction (\Cref{d:mpstAPCP:msgProp}) and the predicate $\hdep$ (\Cref{d:mpstAPCP:hdep}).
\end{definition}

\noindent
We consider the local projection of an exchange in a global type onto a participant $p$ with priority $\pri$.
The priorities in local projection reflect the four sub-steps into which we decompose exchanges in global types (cf.\ \Cref{s:mpstAPCP:analysis:routers}).
There are three possibilities, depending on the involvement of $p$ in the exchange:
\begin{enumerate}
    \item
        If $p$ is the sender, local projection specifies a selection ($\oplus$) between the exchange's labels at priority $\pri$ and a send ($\tensor$) of the associated message type at priority $\pri+1$, followed by the projection of the chosen branch at priority $\pri+4$.
    \item
        If $p$ is the recipient, local projection specifies a branch ($\&$) on the exchange's labels at priority $\pri+2$ and a receive ($\parr$) of the associated message type at priority $\pri+3$, followed by the projection of the chosen branch at priority $\pri+4$.
    \item
        If $p$ is neither sender nor recipient, local projection uses the predicate $\hdep$ (\Cref{d:mpstAPCP:hdep}) to detect a dependency on the sender's output or the recipient's input.
        If there is a dependency on the output, local projection specifies a branch on the exchange's labels at priority $\pri+2$.
        If there is a dependency on the input, local projection specifies a branch at priority $\pri+3$.
        Otherwise, when there is no dependency at all, local projection simply continues with the projection of any branch at priority~$\pri+4$.
\end{enumerate}
Projection only preserves recursive definitions if they contain actual behavior (i.e., the projection of the recursive loop is  contractive, cf.\ \Cref{d:mpstAPCP:contrSes}).
The projections of $\tEnd$ and recursion variables are homomorphic.
The projection of $\tSkip$ projects the $\tSkip$'s continuation, at priority $\pri+4$ to keep the priority aligned with the priorities of the other types of the router.

\subsubsection{The Channels between Pairs of Routers}
\label{s:mpstAPCP:analysis:typing:relProp}

\begin{figure}[t]
    \ifthesis\else \begin{mdframed} \fi
        \begin{align*}
            \relProp{p}{q}{\pri}{s {!} r \rtBraces{ \msg i<S_i> \rtc. \rtc{R_i} }_{i \in I}}
            &\deq \begin{cases}
                \oplus^{\pri+1} \left\{ i : \msgProp{S_i} \tensor^{\pri+2} \relProp{p}{q}{\pri+4}{\rtc{R_i}} \right\}_{i \in I}
                & \text{if $p = s$}
                \\[6pt]
                \&^{\pri+1} \left\{ i : \ol{\msgProp{S_i}} \parr^{\pri+2} \relProp{p}{q}{\pri+4}{\rtc{R_i}} \right\}_{i \in I}
                & \text{if $q = s$}
            \end{cases}
            \span\span\span\span\span\span
            \\[2pt]
            \relProp{p}{q}{\pri}{(r {?} s) {!} t \rtBraces{ i \rtc. \rtc{R_i} }_{i \in I}}
            &\deq \begin{cases}
                \oplus^{\pri+2} \left\{ i : \relProp{p}{q}{\pri+4}{\rtc{R_i}} \right\}_{i \in I}
                & \text{if $p = r$}
                \\[6pt]
                \&^{\pri+2} \left\{ i : \relProp{p}{q}{\pri+4}{\rtc{R_i}} \right\}_{i \in I}
                & \text{if $q = r$}
            \end{cases}
            \span\span\span\span\span\span
            \\[2pt]
            \relProp{p}{q}{\pri}{(s {!} r) {!} t \rtBraces{ i \rtc. \rtc{R_i} }_{i \in I}}
            &\deq \begin{cases}
                \oplus^{\pri+1} \left\{ i : \relProp{p}{q}{\pri+4}{\rtc{R_i}} \right\}_{i \in I}
                & \text{if $p = s$}
                \\[6pt]
                \&^{\pri+1} \left\{ i : \relProp{p}{q}{\pri+4}{\rtc{R_i}} \right\}_{i \in I}
                & \text{if $q = s$}
            \end{cases}
            \span\span\span\span\span\span
            \\[2pt]
            \relProp{p}{q}{\pri}{\tEnd}
            &\deq \bullet
            &
            \relProp{p}{q}{\pri}{\tSkip \rtc. \rtc{R}}
            &\deq \relProp{p}{q}{\pri+4}{\rtc{R}}
            &
            \relProp{p}{q}{\pri}{\mu X \rtc. \rtc{R}}
            &\deq \mu X . \relProp{p}{q}{\pri}{\rtc{R}}
            &
            \relProp{p}{q}{\pri}{X}
            &\deq X
        \end{align*}
    \ifthesis\else \end{mdframed} \fi

    \caption{Extracting session types from relative types (cf.\ \Cref{d:mpstAPCP:relProp}).}
    \label{f:mpstAPCP:relProp}
\end{figure}

For the channels between pairs of routers (given in \textcolor{rchClr}{\rchClrName}), we extract session types from relative types (\Cref{d:mpstAPCP:relTypes}).
Considering a relative type that describes the protocol between $p$ and $q$, this entails decomposing it into a type for $p$ and a dual type for $q$.

\begin{definition}[From Relative Types to Session Types]
    \label{d:mpstAPCP:relProp}
    We define the extraction of a session type from relative type $\rtc{R}$ between $p$ and $q$ at $p$'s perspective with priority $\pri$, denoted $\relProp{p}{q}{\pri}{\rtc{R}}$, by induction on the structure of $\rtc{R}$ as in \Cref{f:mpstAPCP:relProp}.
\end{definition}

\noindent
Here, extraction is \emph{directional}: in $\relProp{p}{q}{\pri}{\rtc{R}}$, the annotation $p \textcolor{relpropClr}{\>} q$ says that the session type describes the perspective of $p$'s router with respect to $q$'s.
Messages with sender $p$ are decomposed into selection ($\oplus$) at priority $\pri+1$ followed by a send~($\tensor$) at priority $\pri+2$.
Dependencies on messages recieved by $p$ become selection types ($\oplus$) at priority $\pri+1$, and dependencies on messages sent by $p$ become selection types ($\oplus$) at priority $\pri+2$.
Messages from $q$ and dependencies on $q$ yield dual types.
Extraction from $\tEnd$ and recursion is homomorphic, and extraction from $\tSkip$ simply extracts from the $\tSkip$'s continuation at priority $\pri+4$.

This way, the channel endpoint of $p$'s router that connects to $q$'s router will be typed $\relProp{p}{q}{\pri}{\gtc{G} \wrt (p,q)}$, i.e.\ the session type extracted from the relative projection of $\gtc{G}$ onto $p,q$ at $p$'s perspective.
Similarly, the endpoint of this channel at $q$'s router will have the type
$\relProp{q}{p}{\pri}{\gtc{G} \wrt (p,q)}$, i.e.\ the same relative projection but at $q$'s perspective.
Clearly, these session types must be dual.

\begin{theorem}
    \label{t:mpstAPCP:relPropDual}
    Given a relative well-formed global type $\gtc{G}$ and $p,q \in \prt(\gtc{G})$,
    \[
        \relProp{p}{q}{\pri}{\gtc{G} \wrt (p,q)} = \ol{\relProp{q}{p}{\pri}{\gtc{G} \wrt (p,q)}}.
    \]
\end{theorem}

\begin{sketch}
    By construction from \Cref{d:mpstAPCP:relProj,d:mpstAPCP:relProp}.
\end{sketch}

\subsection{Networks of Routed Implementations}
\label{s:mpstAPCP:analysis:networks}

Having defined routers and types for their channels, we now turn to defining \emph{networks of routed implementations}, i.e., process networks of routers and implementations that correctly represent a given multiparty protocol.
Then, we appeal to the types obtained in \Cref{s:mpstAPCP:analysis:typing} to establish the typability of routers (\Cref{t:mpstAPCP:routerTypes}).
Finally, we show that all networks of routed implementations of well-formed global types are deadlock-free (\Cref{t:mpstAPCP:globalDf}), and that networks of routed implementations behave as depicted by the global types from which they are generated (\Cref{t:mpstAPCP:completeness,t:mpstAPCP:soundness}).
\ifthesis
\ifappendix \Cref{ac:mpstAPCP} \else \cite{thesis/vdHeuvel24ex} \fi details proofs (including auxiliary definitions and results, and an overview of them); here, we give only proof sketches.
\fi

We begin by defining routed implementations, which connect implementations of subsets of protocol participants with routers:

\begin{definition}[Routed Implementations]
    \label{d:mpstAPCP:routedImplementations}
    Given a closed, relative well-formed global type $\gtc{G}$,
    for participants $\tilde{p} \subseteq \prt(\gtc{G})$, the \emph{set of routed implementations} of $\tilde{p}$ in $\gtc{G}$ is defined as follows (cf.\ \Cref{d:mpstAPCP:locProj} for local projection `$\,\locProjD\,$' and \Cref{d:mpstAPCP:router} for router synthesis `$\,\lmed{}{}{\ldots}\,$'):
    \[
        \lsys(\gtc{G}, \tilde{p}) \deq \left\{
            \pRes{\ci{\mu}{p} \ci{p}{\mu}}_{p \in \tilde{p}} ( Q \| \prod_{p \in \tilde{p}} \rtr{p}) \mathrel{}\middle\vert\mathrel{} \begin{array}{@{}l@{}}
                \emptyset \vdash Q \typInf \Gamma, {(\ci{p}{\mu}:(\gtc{G} \locProjD^0 p))}_{p \in \tilde{p}}
                \\[6pt]
                {} \wedge \forall p \in \tilde{p}.~ \rtr{p} = \lmed{p}{\prt(\gtc{G}) \setminus \{p\}}{\gtc{G}}
            \end{array}
        \right\}
    \]
    We write $\mcl{N}_{\tilde{p}}, \mcl{N}'_{\tilde{p}}, \ldots$ to denote elements of $\lsys(G, \tilde{p})$.
\end{definition}

\noindent
Thus, the composition of a collection of routers and an implementation $Q$ is a routed implementation as long as $Q$ can be typed in a context that includes the corresponding projected types.
Note that the parameter $\tilde{p}$ indicates the presence of \emph{interleaving}: when $\tilde{p}$ is a singleton, the set $\lsys(\gtc{G}, \tilde{p})$ contains processes in which there is a single router and the implementation $Q$ is single-threaded (non-interleaved); more interestingly, when $\tilde{p}$ includes two or more participants, the set $\lsys(\gtc{G}, \tilde{p})$ consists of processes in which the implementation $Q$ {interleaves} the roles of the multiple participants in $\tilde{p}$.

A network of routed implementations of a global type, or simply a \emph{network}, is then the composition of any combination of routed implementations that together account for all the protocol's participants.
Hence, we define sets of networks, quantified over all possible combinations of sets of participants and their respective routed implementations.
The definition relies on \emph{complete partitions} of the participants of a global type, i.e., a split of $\prt(\gtc{G})$ into non-empty, disjoint subsets whose union yields $\prt(\gtc{G})$.

\begin{definition}[Networks]
    \label{d:mpstAPCP:networks}
    Suppose given a closed, relative well-formed global type $\gtc{G}$.
    Let $\mbb{P}_{\gtc{G}}$ be the set of all complete partitions of $\prt(\gtc{G})$ with elements $\pi, \pi', \ldots$
    The \emph{set of networks} of $\gtc{G}$ is defined as
    \[
        \sys(\gtc{G}) \deq \big\{
            \pRes{\crt{p}{q} \crt{q}{p}}_{p,q \in \prt(\gtc{G})} ( \prod_{\tilde{p} \in \pi} \mcl{N}_{\tilde{p}}) ~\big\vert~  \pi \in \mbb{P}_{\gtc{G}} \wedge \forall \tilde{p} \in \pi.~ \mcl{N}_{\tilde{p}} \in \lsys(\gtc{G},\tilde{p})
        \big\}.
    \]
    We write $\mcl{N}, \mcl{N}', \ldots$ to denote elements of $\sys(\gtc{G})$.
\end{definition}

\begin{example}
    \Cref{f:mpstAPCP:GauthNetwork} depicts two networks in $\sys(\gtc{G_{\sff{auth}}})$ related to different partitions of $\prt(\gtc{G_{\sff{auth}}})$, namely $\big\{ \{a\}, \{s\}, \{c\} \big\}$ (non-interleaved) on the left and $\big \{\{a,s\}, \{c\} \big\}$ (interleaved) on the right.
\end{example}

Because a network $\mcl{N}$ may not be typable under the empty typing context, we have the following definition to ``complete'' networks.

\begin{definition}[Completable Networks]
    \label{d:mpstAPCP:complNet}
    Suppose given a network $\mcl{N}$ such that \mbox{$\emptyset \vdash \mcl{N} \typInf \Gamma$}.
    We say that $\mcl{N}$ is \emph{completable} if (i)~$\Gamma$ is empty or (ii)~there exist $\tilde{v},\tilde{w}$ such that $\emptyset \vdash \pRes{\tilde{v} \tilde{w}} \mcl{N} \typInf \emptyset$.
    When $\mcl{N}$ is completable, we write $\mcl{N}^\complet$ to stand for $\mcl{N}$ (if $\emptyset \vdash \mcl{N} \typInf \emptyset$) or $\pRes{\tilde{v} \tilde{w}} \mcl{N}$ (otherwise).
\end{definition}

\begin{restatable}{proposition}{pMpstAPCPComplNetsExist}
    \label{p:mpstAPCP:complNetsExist}
    For any closed, relative well-formed global type $\gtc{G}$, there exists at least one completable network $\mcl{N} \in \sys(\gtc{G})$.
\end{restatable}

\ifthesis
\begin{sketch}
    By generating processes typed $\emptyset \vdash Q \typInf \ci{p}{\mu}:(\gtc{G} \locProjD^0 p)$ and composing them with routers appropriately.
\end{sketch}
\else
\begin{proof}
    To construct a completable network in $\sys(\gtc{G})$, we construct a routed implementation (\Cref{d:mpstAPCP:routedImplementations}) for every $p \in \prt(\gtc{G})$.
    Given a $p \in \prt(\gtc{G})$, by \Cref{p:mpstAPCP:charProc}, there exists $\emptyset \vdash Q \typInf \ci{p}{\mu}:(\gtc{G} \locProjD^0 p)$.
    Composing each such characteristic implementation process with routers, and then composing the routed implementations, we obtain a network $\mcl{N} \in \sys(\gtc{G})$, where $\emptyset \vdash \mcl{N} \typInf \emptyset$.
    Hence, $\mcl{N}$ is completable.
\end{proof}
\fi

\subsubsection{The Typability of Routers}
\label{s:mpstAPCP:analysis:networks:routerTypes}

\begin{figure}[t]
    \ifthesis\else \begin{mdframed} \fi
        \centering
        \ifthesis \begin{tikzpicture}
    \begin{scope}[local bounding box=bndP]
        \node (P) [draw=black, shape=rectangle, inner sep=1mm] at (0, 0) {$P$};
        \node (struP) [right=2cm of P, yshift=6mm] {};
        \node (strdP) [right=2cm of P, yshift=-2mm] {};
    \end{scope}
    \node [fit=(bndP), shape=rectangle, inner sep=0mm] (boxP) {};
    \node [below=6mm of P, anchor=north west, xshift=-2mm] (Ptxt) {implementation};
    \draw[-to, cblRed] (Ptxt.north west) ++ (right:2mm) -- (P);

    \begin{scope}[local bounding box=bndRtr]
        \node (rtr) [right=5cm of P, draw=black, shape=rectangle, inner sep=1mm] {$\lmed{c}{\{s,a\}}{\gtc{G_{\sff{auth}}}}$};
        \node (struRtr) [left=2cm of rtr, yshift=5mm] {};
        \node (strdRtr) [left=2cm of rtr, yshift=-3mm] {};
        \node (strS) [right=3.1cm of rtr, yshift=9mm] {};
        \node (strA) [right=3.1cm of rtr, yshift=-10mm] {};
    \end{scope}
    \node [fit=(bndRtr), shape=rectangle, inner sep=0mm] (boxRtr) {};
    \node [below=5.5mm of rtr] (rtrTxt) {router};
    \draw[-to, cblRed] (rtrTxt) -- (rtr);

    \draw [-, ichClr] (P) -- (rtr) node (cmu) [at start, above=-.5mm, xshift=1.1cm] {$\ci{c}{\mu} \color{black} : (\gtc{G_{\sff{auth}}} \locProjD^\pri c)$} node (muc) [at end, above=-.5mm, xshift=-1.1cm] {$\ci{\mu}{c} \color{black} : \ol{(\gtc{G_{\sff{auth}}} \locProjD^\pri c)}$};

    \node (s) [right=3.5cm of rtr, yshift=5mm, inner sep=.5mm] {};
    \node (s') [right=1cm of s, yshift=1.4mm, inner sep=.5mm] {};

    \draw[-, rchClr] (rtr) -- (s) node (cs) [pos=.45, sloped, above=-1mm] {$\crt{c}{s} \color{black} : \relProp{c}{s}{\pri}{\gtc{G_{\sff{auth}}} \wrt (c,s)}$};
    \draw[loosely dotted, rchClr] (s) -- (s');

    \node (a) [right=3.5cm of rtr, yshift=-5mm, inner sep=.5mm] {};
    \node (a') [right=1cm of a, yshift=-1.4mm, inner sep=.5mm] {};

    \draw[-, rchClr] (rtr) -- (a) node [pos=.45, sloped, below=0mm] {$\crt{c}{a} \color{black} : \relProp{c}{a}{\pri}{\gtc{G_{\sff{auth}}} \wrt (c,a)}$};
    \draw[loosely dotted, rchClr] (a) -- (a');

    \node (locproj) [above=1cm of cmu, xshift=1.5cm] {\defref{d:mpstAPCP:locProj}};
    \draw[-to, cblRed] (locproj) -- (cmu);
    \draw[-to, cblRed] (locproj) -- (muc);

    \node (lmed) [above=1.1cm of rtr] {\defref{d:mpstAPCP:router}};
    \draw[-to, cblRed] (lmed) -- (rtr);

    \node (relproj) [above=0.6cm of cs] {\defref{d:mpstAPCP:relProj}};
    \draw[-to, cblRed] (relproj) -- (cs);

    \node (relprop) [above=0.6cm of cs, xshift=2.5cm] {\defref{d:mpstAPCP:relProp}};
    \coordinate[right=0mm of cs, yshift=1mm, xshift=-1mm] (csc);
    \draw[-to, cblRed] (relprop) -- (csc);
\end{tikzpicture}
        \else \begin{tikzpicture}
    \begin{scope}[local bounding box=bndP]
        \node (P) [draw=black, shape=rectangle, inner sep=1mm] at (0, 0) {$P$};
        \node (struP) [right=2cm of P, yshift=6mm] {};
        \node (strdP) [right=2cm of P, yshift=-2mm] {};
    \end{scope}
    \node [fit=(bndP), shape=rectangle, inner sep=0mm] (boxP) {};
    \node [below=6mm of P, anchor=north west, xshift=-2mm] (Ptxt) {implementation};
    \draw[-to, cblRed] (Ptxt.north west) ++ (right:2mm) -- (P);

    \begin{scope}[local bounding box=bndRtr]
        \node (rtr) [right=5cm of P, draw=black, shape=rectangle, inner sep=1mm] {$\lmed{c}{\{s,a\}}{\gtc{G_{\sff{auth}}}}$};
        \node (struRtr) [left=2cm of rtr, yshift=5mm] {};
        \node (strdRtr) [left=2cm of rtr, yshift=-3mm] {};
        \node (strS) [right=3.1cm of rtr, yshift=9mm] {};
        \node (strA) [right=3.1cm of rtr, yshift=-10mm] {};
    \end{scope}
    \node [fit=(bndRtr), shape=rectangle, inner sep=0mm] (boxRtr) {};
    \node [below=5.5mm of rtr] (rtrTxt) {router};
    \draw[-to, cblRed] (rtrTxt) -- (rtr);

    \draw [-, ichClr] (P) -- (rtr) node (cmu) [at start, above=-.5mm, xshift=1.1cm] {$\ci{c}{\mu} \color{black} : (\gtc{G_{\sff{auth}}} \locProjD^\pri c)$} node (muc) [at end, above=-.5mm, xshift=-1.1cm] {$\ci{\mu}{c} \color{black} : \ol{(\gtc{G_{\sff{auth}}} \locProjD^\pri c)}$};

    \node (s) [right=3.5cm of rtr, yshift=5mm, inner sep=.5mm] {};
    \node (s') [right=1cm of s, yshift=1.4mm, inner sep=.5mm] {};

    \draw[-, rchClr] (rtr) -- (s) node (cs) [pos=.45, sloped, above=-1mm] {$\crt{c}{s} \color{black} : \relProp{c}{s}{\pri}{\gtc{G_{\sff{auth}}} \wrt (c,s)}$};
    \draw[loosely dotted, rchClr] (s) -- (s');

    \node (a) [right=3.5cm of rtr, yshift=-5mm, inner sep=.5mm] {};
    \node (a') [right=1cm of a, yshift=-1.4mm, inner sep=.5mm] {};

    \draw[-, rchClr] (rtr) -- (a) node [pos=.45, sloped, below=0mm] {$\crt{c}{a} \color{black} : \relProp{c}{a}{\pri}{\gtc{G_{\sff{auth}}} \wrt (c,a)}$};
    \draw[loosely dotted, rchClr] (a) -- (a');

    \node (locproj) [above=1cm of cmu, xshift=1.5cm] {\defref{d:mpstAPCP:locProj}};
    \draw[-to, cblRed] (locproj) -- (cmu);
    \draw[-to, cblRed] (locproj) -- (muc);

    \node (lmed) [above=1.1cm of rtr] {\defref{d:mpstAPCP:router}};
    \draw[-to, cblRed] (lmed) -- (rtr);

    \node (relproj) [above=0.6cm of cs] {\defref{d:mpstAPCP:relProj}};
    \draw[-to, cblRed] (relproj) -- (cs);

    \node (relprop) [above=0.6cm of cs, xshift=2.5cm] {\defref{d:mpstAPCP:relProp}};
    \coordinate[right=0mm of cs, yshift=1mm, xshift=-1mm] (csc);
    \draw[-to, cblRed] (relprop) -- (csc);
\end{tikzpicture} \fi
    \ifthesis\else \end{mdframed} \fi

    \caption{
        Overview of \Cref{t:mpstAPCP:routerTypes}, with the definitions and notations for synthesizing and typing routers, using participant $c$ of $\gtc{G_{\sff{auth}}}$ implemented as $P$ (cf. \Cref{x:mpstAPCP:impl}).
        Boxes indicate processes and lines indicate channels.
    }
    \label{f:mpstAPCP:typedRouter}
\end{figure}

We wish to establish that the networks of a global type are deadlock-free.
This result, formalized by \Cref{t:mpstAPCP:globalDf} (\Cpageref{t:mpstAPCP:globalDf}), hinges on the typability of routers, which we address next.
\Cref{f:mpstAPCP:typedRouter} gives an overview of the definitions and notations involved.

\begin{restatable}{theorem}{tMpstAPCPRouterTypes}
    \label{t:mpstAPCP:routerTypes}
    Suppose given a closed, relative well-formed global type $\gtc{G}$, and a $p \in \prt(\gtc{G})$.
    Then,
    \[
        \emptyset \vdash \lmed{p}{\prt(\gtc{G}) \setminus \{p\}}{\gtc{G}} \typInf \ci{\mu}{p}:\ol{\gtc{G} \locProjD^0 p}, {\big(\crt{p}{q}:\relProp{p}{q}{0}{\gtc{G} \wrt (p,q)}\big)}_{q \in \prt(\gtc{G}) \setminus \{p\}}.
    \]
\end{restatable}

\noindent
\ifmpstAPCP
\Cref{t:mpstAPCP:routerTypes} follows as a corollary of a more general statement, enabling induction on the structure of $\gtc{G}$ as well as reasoning about recursive types.
The latter requires extensive knowledge about the context of recursive definitions in $\gtc{G}$.
\ifappendix \Cref{as:mpstAPCP:analysis:networks:routerTypes} \else \cite{thesis/vdHeuvel24ex} \fi details the auxiliary definitions and results needed for this, as well as a complete proof of \Cref{t:mpstAPCP:routerTypes}.
There is another important detail about the typability of routers that we need to address in detail here.
\else
This result is a corollary of \Cref{t:mpstAPCP:routerTypesGen} (\Cpageref{t:mpstAPCP:routerTypesGen}), which we show next.
We give a full proof on \Cpageref{proof:mpstAPCP:routerTypes}, after the proof of \Cref{t:mpstAPCP:routerTypesGen}.
\fi

\paragraph{Alarm processes.}

We focus on networks of routed implementations---compositions of synthesized routers and well-typed processes.
However, in order to establish the typability of routers we must account for an edge case that goes beyond these assumptions, namely when a routed implementation is connected to some undesirable implementation, not synthesized by \Cref{alg:mpstAPCP:router}.
Consider the following example:

\begin{example}
    \label{x:mpstAPCP:alarm}
    Consider again the global type $\gtc{G_{\sff{auth}}}$, which, for the purpose of this example, we write as follows:
    \[
        \gtc{G_{\sff{auth}}} = s {!} c \gtBraces*{
            \begin{array}{@{}l@{}}
                \sff{login} : \gtc{G_{\sff{login}}} \gtc,
                \\
                \sff{quit} : \gtc{G_{\sff{quit}}}
            \end{array}
        }
    \]
    As established in \Cref{x:mpstAPCP:GauthProj}, the initial exchange between $s$ and $c$ determines a dependency for the interactions of $a$ with both $s$ and $c$.
    Therefore, the implementation of $a$ needs to receive the choice between \sff{login} and \sff{quit} from the implementations of both $s$ and $c$.
    An undesirable implementation for $c$, without a router, could be for instance as follows:
    \[
        R' \deq \pBra* \crt{c}{s} > {
            \left\{
                \begin{array}{@{}l@{}}
                    \sff{login} : \pSel* \crt{c}{a} < \sff{quit} \cdot \ldots,
                    \\
                    \sff{quit} : \pSel* \crt{c}{a} < \sff{quit} \cdot \ldots
                \end{array}
            \right\}
        }
    \]
    Notice how $ R'$ always sends to $a$ the label \sff{quit}, even if the choice made by $s$ (and sent to $c$) is \sff{login}.
    Now, if $s$ chooses \sff{login}, the router of $a$ is in limbo: on the one hand, it expects $s$ to behave as specified in $\gtc{G_{\sff{login}}}$; on the other hand, it expects $c$ to behave as specified in $\gtc{G_{\sff{quit}}}$.
    Clearly, the router of $a$ is in an inconsistent state due to $c$'s implementation.
\end{example}

\noindent
Because routers always forward the chosen label correctly, this kind of undesirable behavior never occurs in the networks of \Cref{d:mpstAPCP:networks}---we state this formally in \Cref{s:mpstAPCP:analysis:transfer} (\Cref{t:mpstAPCP:noAlarm}).
Still, in order to prove that our routers are well-typed, we must accommodate the possibility that a router ends up in an undesirable state due to inconsistent forwarding.
For this, we extend \APCP with an \emph{alarm process} that signals an inconsistency on a given set of channel endpoints.

\begin{definition}[Alarm Process]
    \label{d:mpstAPCP:alarm}
    Given channel endpoints $\tilde{x} = x_1, \ldots, x_n$, we write $\pAlarm(\tilde{x})$ to denote an inconsistent state on those endpoints.
\end{definition}

\noindent
In a way, $\pAlarm(\tilde{x})$ is closer to an observable action (a ``barb'') than to an actual process term: $\pAlarm(\tilde{x})$ does not have reductions, and no \ifappendix process from \Cref{d:mpstAPCP:APCP:syntax} \else \APCP process \fi can reduce to $\pAlarm(\tilde{x})$.
We assume that $\pAlarm(\tilde{x})$ does not occur in participant implementations (cf.\ $Q$ in \Cref{d:mpstAPCP:routedImplementations});
we treat it as a process solely for the purpose of refining the router synthesis algorithm (\Cref{alg:mpstAPCP:router}) with the possibility of inconsistent forwarding.
The refinement concerns the process on \cref{li:mpstAPCP:rtrBoth}:
\[
    \pBra* \crt{p}{s} > {\big\{ i : \pSel* \ci{\mu}{p} < i \cdot \pBra* \crt{p}{r} > {\{ i: \lmed{p}{\tilde{q}}{\gtc{G_i}} \}} \big\}_{i \in I}}
\]
We extend it with additional branches, as follows:
\begin{align}
    \pBra* \crt{p}{s} > {
        \big\{
            i : \pSel* \ci{\mu}{p}} < i \cdot \pBra* \crt{p}{r} > {
                \left(
                    \begin{array}{l}
                        \{ i : \lmed{p}{\tilde{q}}{\gtc{G_i}} \}
                        \\[5pt]
                        \underline{\cup~ \{ i' : \pAlarm(\ci{\mu}{p}, {(\crt{p}{q})}_{q \in \tilde{q}}) \}_{i' \in I \setminus \{i\}}}
                    \end{array}
                \right)
            \big\}_{i \in I}
        }
    \label{eq:mpstAPCP:patched}
\end{align}
This new process for \cref{li:mpstAPCP:rtrBoth} captures the kind of inconsistency illustrated by \Cref{x:mpstAPCP:alarm}, which occurs when a label $i \in I$ is received over $\crt{p}{s}$ after which a label $i' \in I \setminus \{i\}$ is received over $\crt{p}{r}$.
We account for this case by using the underlined alarm processes.

Routers are then made of \ifappendix processes as in \Cref{d:mpstAPCP:APCP:syntax}\else \APCP processes\fi, selectively extended with alarms as just described.
Because $\pAlarm(\tilde{x})$ merely acts as an observable that signals undesirable behavior, we find it convenient to type it using the following axiom:
\[
    \begin{bussproof}[typ-alarm]
        \bussAx{
            \Omega \vdash \pAlarm(x_1, \ldots, x_n) \typInf x_1:A_1, \ldots, x_n:A_n
        }
    \end{bussproof}
\]
where the recursive context $\Omega$ and types $A_1, \ldots, A_n$ are arbitrary.

\ifmpstAPCP \else
\paragraph{Context-based typability.}

Considering the refinement of \Cref{alg:mpstAPCP:router} with alarm processes, we prove \Cref{t:mpstAPCP:routerTypesGen} on \Cpageref{t:mpstAPCP:routerTypesGen}, from which \Cref{t:mpstAPCP:routerTypes} follows as a corollary.
It relies on some additional auxiliary definitions and results.

\begin{itemize}
    \item
        To type the router for a participant at any point in the protocol, we need the definition of the entire protocol.
        It is not enough to only consider the current (partial) protocol at such points: we need information about bound recursion variables in order to perform unfolding in types.
        To this end, \Cref{d:mpstAPCP:globCtx} defines \emph{global contexts}, that allow us to look at part of a protocol while retaining definitions that concern the entire protocol.

        As mentioned before, a context captures information about the recursion variables that are bound at any given point in a global type.
        Our goal is to obtain a \emph{context-based} typability result for routers.

    \item
        The order in which recursive variables are bound is important to correctly unfold types, as illustrated by \Cref{x:mpstAPCP:unfolding}.
        To account for nested recursions, \Cref{d:mpstAPCP:binders} gives the bound variables of a context exactly in the order in which they appear.

    \item
        \Cref{d:mpstAPCP:recDef} retrieves the body of a recursive definition from a global context, informing us on how to unfold types.

    \item
        When unfolding bound recursion variables, we need the priorities of the unfolded types.
        \Cref{d:mpstAPCP:ctxPri} gives a priority that is expected at the hole in a context, as well as the priority expected at any recursive definition in a global type:

    \item
        To avoid non-contractive recursive types, relative projection (cf.\ \Cref{alg:mpstAPCP:relProj}) closes a type when the participants do not interact inside a recursive definition.
        Hence, when typing a router for a recursive definition, we must determine which pairs of participants are ``active'' at any given point in a protocol, and close the connections with the ``inactive'' participants.
        See \Cref{x:mpstAPCP:active} for an illustration.

        \Cref{d:mpstAPCP:activeParts} uses relative projection to determine the pairs of active participants at the hole of a context, as well as at any recursive definition in a global type.
        We consider pairs of participants $(p,q)$ and $(q,p)$ to be equivalent.

    \item
        When typing a router for a given protocol, we have to keep track of assignments in the recursive context at any point in the protocol.
        \Cref{l:mpstAPCP:binderCtxSubRec,l:mpstAPCP:binderParts} ensure that the active participants of recursive definitions are consistent with the active participants of their bodies.
        \Cref{l:mpstAPCP:binderParts} then ensures that when typing a recursive call, the endpoints given as context for the recursive call concur with the endpoints in the recursive context.

    \item
        Our typability result for routers relies on relative and local projection.
        Hence, we need to guarantee that all the projections we need at any given point of a protocol are defined.
        \Cref{p:mpstAPCP:subWF} shows a form of compositionality for relative and local projection, guaranteeing the definedness of projections for all active participants of a given context:

    \item
        Recall \Cref{x:mpstAPCP:unfolding}, where nested recursive definitions in a protocol require nested unfolding of recursive types.
        \Cref{d:mpstAPCP:deepUnfold} gives us a concise way of writing such nested (or \emph{deep}) unfoldings:

    \item
        When typing a router's recursive call, the types of the router's endpoints are unfoldings of the types in the recursive context.
        However, because of the deep unfolding in types, this is far from obvious.
        \Cref{p:mpstAPCP:deepUnfold} connects a particular form of deep unfolding with regular unfolding.

    \item
        Armed with these definitions and results, \Cref{t:mpstAPCP:routerTypesGen} states our context-based typability result for routers.

    \item
        Finally, we repeat and prove \Cref{t:mpstAPCP:routerTypes} as a corollary of \Cref{t:mpstAPCP:routerTypesGen}, on \Cpageref{proof:mpstAPCP:routerTypes}.
\end{itemize}

\begin{definition}[Global Contexts]
    \label{d:mpstAPCP:globCtx}
    \emph{Global contexts} $\globCtx{C}$ are given by the following grammar:
    \[
        \globCtx{C} \deq p {!} q ( \gtBraces{ \msg i<S> \gtc. \gtc{G} }_{i \in I} \cup \gtBraces{ \msg i'<S> \gtc. \globCtx{C} }_{i' \notin I}
        \sepr \tSkip \gtc. \globCtx{C}
        \sepr \mu X \gtc. \globCtx{C}
        \sepr \globHole
    \]
    We often simply write ``context'' when it is clear that we are referring to a global context.
    Given a context $\globCtx{C}$ and a global type $\gtc{G}$, we write $\globCtx{C}[\gtc{G}]$ to denote the global type obtained by replacing the hole $\globHole$ in $\globCtx{C}$ with $\gtc{G}$.
    If $\gtc{G} = \globCtx{C}[\gtc{G_s}]$ for some context $\globCtx{C}$ and global type $\gtc{G_s}$, then we write $\gtc{G_s} \gSub{\globCtx{C}} \gtc{G}$.
\end{definition}

\begin{example}
    \label{x:mpstAPCP:unfolding}
    Consider the following global type with three nested recursive definitions:
    \[
        \gtc{G_{\sff{rec}}} = \mu X \gtc. a {!} b (1) \gtc. \mu Y \gtc. a {!} b (2) \gtc. \mu Z \gtc. a {!} b \gtBraces{ \sff{x} : X \gtc, \sff{y} : Y \gtc, \sff{z} : Z }
    \]
    To type the router for, e.g., $a$ at the final exchange between $a$ and $b$, we need to be aware of the unfolding of recursion.
    The recursion on $X$, $Y$, and $Z$ have all to be unfolded, and the recursion on $Z$ must include first the unfolding of $X$ and then the unfolding of $Y$, which must in turn include the prior unfolding of~$X$.
\end{example}

\begin{definition}[Recursion Binders of Contexts]
    \label{d:mpstAPCP:binders}
    Given a global context $\globCtx{C}$, the \emph{sequence of recursion binders to the hole} of $\globCtx{C}$, denoted $\ctxBind(\globCtx{C})$, is defined as follows:
    \begin{align*}
        \ctxBind(\mu X \gtc. \globCtx{C})
        &\deq (X, \ctxBind(\globCtx{C}))
        &
        \ctxBind(\tSkip \gtc. \globCtx{C})
        &\deq \ctxBind(\globCtx{C})
        &
        \ctxBind(\globHole)
        &\deq ()
        \\
        \span\span
        \ctxBind(p {!} q ( \gtBraces{ \msg i<S_i> \gtc. \gtc{G_i} }_{i \in I} \cup \gtBraces{ \msg i'<S_{i'}> \gtc. \globCtx{C} }_{i' \notin I}))
        &\deq \ctxBind(\globCtx{C})
        \span\span
    \end{align*}
    Given $\gtc{G_s} \gSub{\globCtx{C}} \gtc{G}$, the sequence of recursion binders of $\gtc{G_s}$, denoted $\subBind(\gtc{G_s},\gtc{G})$, is defined as $\ctxBind(\globCtx{C})$.
\end{definition}

\begin{definition}[Recursion Extraction]
    \label{d:mpstAPCP:recDef}
    The function $\recDef(X,\gtc{G})$ extracts the recursive definition on $X$ from $\gtc{G}$: $\recDef(X,\gtc{G}) = \gtc{G'}$ if $\mu X \gtc. \gtc{G'} \gSub{\globCtx{C}} \gtc{G}$ for some context $\globCtx{C}$.
    Also, $\recCtx(X,\gtc{G})$ extracts the context of the recursive definition on $X$ in $\gtc{G}$: $\recCtx(X,\gtc{G}) = \globCtx{C}$ if \mbox{$\mu X \gtc. \recDef(X,\gtc{G}) \gSub{\globCtx{C}} \gtc{G}$}.
\end{definition}

\begin{definition}[Absolute Priorities of Contexts]
    \label{d:mpstAPCP:ctxPri}
    Given a context $\globCtx{C}$ and $\pri \in \mbb{N}$, we define $\ctxPri[\pri](C)$ as follows:
    \begin{align*}
        \ctxPri[\pri](\globHole)
        &\deq \pri
        &
        \ctxPri[\pri](\tSkip \gtc. \globCtx{C})
        &\deq \ctxPri[\pri+4](\globCtx{C})
        &
        \ctxPri[\pri](\mu X \gtc. \globCtx{C})
        &\deq \ctxPri[\pri](\globCtx{C})
        \\
        \span\span
        \ctxPri[\pri](p {!} q ( \gtBraces{ \msg i<S_i> \gtc. \gtc{G_i} }_{i \in I} \cup \gtBraces{ \msg i'<S_{i'}> \gtc. \globCtx{C} }_{i' \notin I} ))
        &\deq \ctxPri[\pri+4](\globCtx{C})
        \span\span
    \end{align*}
    Then, the \emph{absolute priority} of $\globCtx{C}$, denoted $\ctxPri(\globCtx{C})$, is defined as $\ctxPri[0](\globCtx{C})$.
    The absolute priority of $X$ in $\gtc{G}$, denoted $\varPri(X,\gtc{G})$, is defined as $\ctxPri(\globCtx{C})$ for some context $\globCtx{C}$ such that $\mu X \gtc. \recDef(X,\gtc{G}) \gSub{\globCtx{C}} \gtc{G}$.
\end{definition}

\begin{example}
    \label{x:mpstAPCP:active}
    Consider the following global type, where a client ($c$) requests two independent, infinite Fibonacci sequences ($f_1$ and $f_2$):
    \begin{align*}
        \gtc{G_{\sff{fib}}} \deq c {!} f_1 (\msg \sff{init}<\sff{int} \times \sff{int}>) \gtc. c {!} f_2 (\msg \sff{init}<\sff{int} \times \sff{int}>) \gtc. \underbrace{\mu X \gtc. f_1 {!} c (\msg \sff{next}<\sff{int}>) \gtc. f_2 {!} c (\msg \sff{next}<\sff{int}>) \gtc. X}_{\gtc{G'_{\sff{fib}}}}
    \end{align*}
    Participants $f_1$ and $f_2$ do not interact with each other in the body of the recursion, as formalized by their relative projection:
    \[
        \recDef(X,\gtc{G_{\sff{fib}}}) \wrt (f_1,f_2) = \tSkip^2 \rtc. X
    \]
    Hence, $\gtc{G'_{\sff{fib}}} \wrt (f_1,f_2) = \tEnd$, and $f_1$ and $f_2$ do not form an active pair of participants for the recursion in $\gtc{G_{\sff{fib}}}$.
    Therefore, $f_1$'s router closes its connection with $f_2$'s router at the start of the recursion on $X$, and vice versa.
\end{example}

\begin{definition}[Active Participants]
    \label{d:mpstAPCP:activeParts}
    Suppose given a relative well-formed global type $\gtc{G}$.
    The following mutually defined functions compute sets of \emph{pairs of active participants} for recursive definitions and contexts, denoted $\recAct(X,\gtc{G})$ and $\act(\globCtx{C},\gtc{G})$, respectively.
    \begin{align*}
        \recAct(X,\gtc{G})
        &\deq \{ (p,q) \in \act(\recCtx(X,\gtc{G}),\gtc{G}) \mid (\mu X \gtc. \recDef(X,\gtc{G})) \wrt (p,q) \neq \tEnd \}
        \\
        \act(\globCtx{C},\gtc{G})
        &\deq \begin{cases}
            \recAct(Y,\gtc{G})
            & \text{if $\ctxBind(\globCtx{C}) = (\tilde{X},Y)$}
            \\
            \prt(\gtc{G})^2
            & \text{otherwise}
        \end{cases}
    \end{align*}
    The interdependency between $\recAct(X,\gtc{G})$ and $\act(\globCtx{C},\gtc{G})$ is well-defined: the former function considers the active participants of a context, which contains less recursive definitions.
\end{definition}

\begin{lemma}
    \label{l:mpstAPCP:binderCtxSubRec}
    Suppose given a closed, relative well-formed global type $\gtc{G}$, and a global type~$\gtc{G_s}$ and context $\globCtx{C}$ such that $\gtc{G_s} \gSub{\globCtx{C}} \gtc{G}$.
    For any $Z \in \ctxBind(\globCtx{C})$, $\act(\globCtx{C},\gtc{G}) \subseteq \recAct(Z,\gtc{G})$.
\end{lemma}

\begin{proof}
    Take any $Z \in \ctxBind(\globCtx{C})$.
    Then $\ctxBind(\globCtx{C}) = (\tilde{X},Y)$.
    By definition,
    \[
        \act(\globCtx{C},\gtc{G}) = \recAct(Y,\gtc{G}).
    \]
    If $Y = Z$, the thesis is proven.
    Otherwise, by definition,
    \[
        {\recAct(Y,\gtc{G}) \subseteq \act(\recCtx(Y,\gtc{G}),\gtc{G})}.
    \]
    Since the recursive definition on $Z$ appears in $\recCtx(Y,\gtc{G})$, it follows by induction on the size of $\tilde{X}$ that \mbox{$\act(\recCtx(Y,\gtc{G}),\gtc{G}) \subseteq \recAct(Z,\gtc{G})$}.
    This proves the thesis.
\end{proof}

\begin{lemma}
    \label{l:mpstAPCP:binderParts}
    Suppose given a closed, relative well-formed global type $\gtc{G}$, a recursion variable~$Z$, and a context $\globCtx{C}$ such that $Z \gSub{\globCtx{C}} \gtc{G}$.
    Then, $\act(\globCtx{C},\gtc{G}) = \recAct(Z,\gtc{G})$.
\end{lemma}

\begin{proof}
    Because $\gtc{G} = \globCtx{C}[Z]$ and $\gtc{G}$ is closed (i.e.\ $\frv(\gtc{G}) = \emptyset$), there is a recursive definition on~$Z$ in $\gtc{G}$.
    Hence, $\ctxBind(\globCtx{C}) \neq \emptyset$, i.e.\ $\ctxBind(\globCtx{C}) = (\tilde{X}, Y)$ and \mbox{$\act(\globCtx{C},\gtc{G}) = \recAct(Y,\gtc{G})$}.
    If $Y = Z$, the thesis is proven.
    Otherwise, the recursive definition on $Y$ in $\gtc{G}$ appears somewhere inside the recursive definition on $Z$.
    Suppose, for contradiction, that \mbox{$\act(\globCtx{C},\gtc{G}) \neq \recAct(Z,\gtc{G})$}.
    There are two cases: there exists $(p,q) \in {\prt(\gtc{G})}^2$ s.t.\ (i)~$(p,q) \in \act(\globCtx{C},\gtc{G})$ and $(p,q) \notin \recAct(Z,\gtc{G})$, or (ii)~$(p,q) \in \recAct(Z,\gtc{G})$ and \mbox{$(p,q) \notin \act(\globCtx{C},\gtc{G})$}.
    Case~(i) contradicts \Cref{l:mpstAPCP:binderCtxSubRec}.

    In case~(ii), $(\mu Z \gtc. \recDef(Z,\gtc{G})) \wrt (p,q) \neq \tEnd$ and $(\mu Y \gtc. \recDef(Y,\gtc{G})) \wrt (p,q) = \tEnd$.
    The recursive call on $Z$ in $\gtc{G}$ appears somewhere inside the recursive definition on $Y$, and hence $\recDef(Y,\gtc{G}) \wrt (p,q)$ contains the recursive call on $Z$.
    This means that $\recDef(Y,\gtc{G}) \wrt (p,q)$ is contractive on $Y$ (\Cref{d:mpstAPCP:relProj:contr}), and hence $(\mu Y \gtc. \recDef(Y,\gtc{G})) \wrt (p,q) \neq \tEnd$, contradicting the assumption.
\end{proof}

\begin{proposition}
    \label{p:mpstAPCP:subWF}
    Suppose given a closed, relative well-formed global type $\gtc{G}$, and a global type $\gtc{G_s}$ such that $\gtc{G_s} \gSub{\globCtx{C}} \gtc{G}$.
    Then, for every $(p,q) \in \act(\globCtx{C},\gtc{G})$, the relative projection \mbox{$\gtc{G_s} \wrt (p,q)$} is defined.
    Also, for every $p \in \{ p \in \prt(\gtc{G}) \mid \exists q \in \prt(G).~ (p,q) \in \act(\globCtx{C},\gtc{G}) \}$, the local projection $\gtc{G_s} \locProjD^\pri p$ is defined for any priority $\pri$.
\end{proposition}

\begin{proof}
    Suppose that, for contradiction, $\gtc{G_s} \wrt (p,q)$ is undefined.
    We show by induction on the structure of $\globCtx{C}$ that this means that $\gtc{G} \wrt (p,q)$ is undefined, contradicting the relative well-formedness of $\gtc{G}$.
    \begin{itemize}

        \item
            Hole: $\globCtx{C} = \globHole$.
            We have $\gtc{G_s} = \gtc{G}$, and the thesis follows immediately.

        \item
            Exchange: $\globCtx{C} = r {!} s ( \gtBraces{ \msg i<S_i> \gtc. \gtc{G_i} }_{i \in I} \cup \gtBraces{ \msg i'<S_{i'}> \gtc. \globCtx{C'} }_{i' \notin I} )$.
            By the IH, $\globCtx{C'}[\gtc{G_s}] \wrt (p,q)$ is undefined.
            Since the relative projection of an exchange relies on the relative projection of each of the exchange's branches, $\gtc{G} \wrt (p,q)$ is undefined.

        \item
            Skip: $\globCtx{C} = \tSkip \gtc. \globCtx{C'}$.
            By the IH, $\globCtx{C'}[\gtc{G_s}] \wrt (p,q)$ is undefined.
            Since the relative projection of a skip relies on the relative projection of the skip's continuation, $\gtc{G} \wrt (p,q)$ is undefined.

        \item
            Recursive definition: $\globCtx{C} = \mu X \gtc. \globCtx{C'}$.
            It follows from \Cref{l:mpstAPCP:binderCtxSubRec} that 
            \[
                \act(\globCtx{C},\gtc{G}) \subseteq \recAct(X,\gtc{G}).
            \]
            Hence, $(p,q) \in \recAct(X,\gtc{G})$, and thus
            \[
                (\mu X \gtc. \recDef(X,\gtc{G})) \wrt (p,q) = (\mu X \gtc. \globCtx{C'}[\gtc{G_s}]) \wrt (p,q) \neq \tEnd,
            \]
            which means that $\globCtx{C'}[\gtc{G_s}] \wrt (p,q)$ is defined.
            This contradicts the IH.

    \end{itemize}
    The proof for the definedness of local projection is analogous.
\end{proof}

\begin{definition}[Deep Unfolding]
    \label{d:mpstAPCP:deepUnfold}
    Suppose given a sequence of tuples $\tilde{U}$, with each tuple consisting of a recursion variable $X_i$, a lift $t_i \in \mbb{N}$, and a type $B_i$.
    The \emph{deep unfolding} of the type~$A$ with $\tilde{U}$, denoted $\deepUnfold(A,\tilde{U})$, is the type defined as follows:
    \begin{align*}
        \deepUnfold(A,())
        &\deq A
        \\
        \deepUnfold(A,(\tilde{U}, (X,t,B)))
        &\deq \deepUnfold(A,\tilde{U}) \{ \big(\mu X . (\lift{t} \deepUnfold(B,\tilde{U}))\big) / X \}
    \end{align*}
\end{definition}

\begin{proposition}
    \label{p:mpstAPCP:deepUnfold}
    Suppose given a type $A$ and a sequence of tuples $\tilde{U}$ consisting of a recursion variable, a lift, and a substitution type.
    Then,
    \[
        \deepUnfold(A,(\tilde{U}, (X, t, A))) = \unfold^t(\mu X . \deepUnfold(A,\tilde{U})).
    \]
\end{proposition}

\begin{proof}
    By \Cref{d:mpstAPCP:deepUnfold}:
    \begin{align*}
        \deepUnfold(A,(\tilde{U}, (X, t, A)))
        &= \deepUnfold(A,\tilde{U}) \{ \big(\mu X . (\lift{t} \deepUnfold(A,\tilde{U}))\big) / X \}
        \\
        &= \unfold^t(\mu X . \deepUnfold(A,\tilde{U}))
        \tag*{\qedhere}
    \end{align*}
\end{proof}

\begin{theorem}
    \label{t:mpstAPCP:routerTypesGen}
    Suppose given a closed, relative well-formed global type $\gtc{G}$.
    Also, suppose given a global type $\gtc{G_s}$ such that $\gtc{G_s} \gSub{\globCtx{C}} \gtc{G}$, and a $p \in \prt(\gtc{G})$ for which there is a $q \in \prt(\gtc{G})$ such that $(p,q) \in \act(\globCtx{C},\gtc{G})$.
    Consider:
    \begin{itemize}

        \item
            the participants with whom $p$ interacts in $\gtc{G_s}$:
            $\tilde{q} \deq \{ q \in \prt(\gtc{G}) \mid (p,q) \in \act(\globCtx{C},\gtc{G}) \}$;

        \item
            the absolute priority of $\gtc{G_s}$:
            $\pri_{\globCtx{C}} \deq \ctxPri(\globCtx{C})$;

        \item
            the sequence of bound recursion variables of $\gtc{G_s}$:
            $\widetilde{X_{\globCtx{C}}} \deq \ctxBind(\globCtx{C})$;

        \item
            for every $X \in \widetilde{X_{\globCtx{C}}}$:
            \begin{itemize}
                \item
                    the body of the recursive definition on $X$ in $\gtc{G}$:
                    $\gtc{G}_X \deq \recDef(X,\gtc{G})$;

                \item
                    the participants with whom $p$ interacts in $\gtc{G}_X$:
                    \[
                        \tilde{q}_X \deq \{ q \in \prt(\gtc{G}) \mid (p,q) \in \recAct(X,\gtc{G}) \};
                    \]

                \item
                    the absolute priority of $\gtc{G}_X$:
                    $\pri_X \deq \varPri(X,\gtc{G})$;

                \item
                    the sequence of bound recursion variables of $\gtc{G}_X$ excluding $X$:
                    \[
                        \widetilde{Y_X} \deq \subBind(\mu X \gtc. \gtc{G}_X,\gtc{G});
                    \]

                \item
                    the type required for $\ci{\mu}{p}$ for a recursive call on $X$:
                    \[
                        A_{X,p} \deq \deepUnfold(\ol{\gtc{G}_X \locProjD^{\pri_X} p},{(Y,t_Y,\ol{\gtc{G}_Y \locProjD^{\pri_Y} p})}_{Y \in \widetilde{Y_X}});
                    \]

                \item
                    the type required for $\crt{p}{q}$ for a recursive call on $X$:
                    \[
                        B_{X,q} \deq \deepUnfold(\relProp{p}{q}{\pri_X}{\gtc{G}_X \wrt (p,q)},{(Y,t_Y,\relProp{p}{q}{\pri_Y}{\gtc{G}_Y \wrt (p,q)})}_{Y \in \widetilde{Y_X}});
                    \]

                \item
                    the minimum lift for typing a recursive definition on $X$:
                    \[
                        t_X \deq \maxpr(A_X, {(B_{X,q})}_{q \in \tilde{q}_X}) + 1;
                    \]

            \end{itemize}

        \item
            the type expected for $\ci{\mu}{p}$ for $p$'s router for $\gtc{G_s}$:
            \[
                D_p \deq \deepUnfold(\ol{\gtc{G_s} \locProjD^{\pri_{\globCtx{C}}} p},{(X,t_X,\ol{\gtc{G}_X \locProjD^{\pri_X} p})}_{X \in \widetilde{X_{\globCtx{C}}}});
            \]

        \item
            the type expected for $\crt{p}{q}$ for $p$'s router for $\gtc{G_s}$:
            \[
                E_q \deq \deepUnfold(\relProp{p}{q}{\pri_{\globCtx{C}}}{\gtc{G_s} \wrt (p,q)},{(X,t_X,\relProp{p}{q}{\pri_X}{\gtc{G}_X \wrt (p,q)})}_{X \in \widetilde{X_{\globCtx{C}}}}).
            \]
    \end{itemize}

    \noindent
    Then, we have:
    \[
        {\Big(X: \big( A_X, {(B_{X,q})}_{q \in \tilde{q}_X} \big) \Big)}_{X \in \widetilde{X_{\globCtx{C}}}}
        \vdash \lmed{p}{\tilde{q}}{\gtc{G_s}}
        \typInf
        \ci{\mu}{p}:D_p,
        {(\crt{p}{q}:E_q)}_{q \in \tilde{q}}
    \]
\end{theorem}

\begin{proof}
    We apply induction on the structure of $\gtc{G_s}$, with six cases as in \Cref{alg:mpstAPCP:router}.
    We only detail the cases of exchange and recursion.
    Rule~\ruleLabel{typ-alarm} is used in only one sub-case (case~3(c), cf.\ \Cref{f:mpstAPCP:tderii} below).

    \begin{itemize}
        \item
            \emph{Exchange}: $\gtc{G_s} = s {!} r \gtBraces{ \msg i<S_i> \gtc. \gtc{G_i} }_{i \in I}$ (\cref{li:mpstAPCP:rtrComm}).

            In this case, we add connectives to the types obtained from the IH.
            Since we do not introduce any recursion variables to these types, the substitutions in the types from the~IH are not affected.
            Hence, we can omit these substitutions from the types.
            Also, for each $i \in I$, we have $\frv(\gtc{G_i}) \subseteq \frv(\gtc{G_s})$, i.e., the recursive context remains untouched in this derivation, so we also omit the recursive context.

            Let $\deps \deq \{ q \in \tilde{q} \mid \hdep(q, p, \gtc{G_s}) \}$ (as on \cref{li:mpstAPCP:rtrDeps}).
            There are three cases depending on the involvement of $p$.
            \begin{enumerate}
                \item
                    If $p = s$, then $p$ is the sender (\cref{li:mpstAPCP:rtrSend}).

                    Let us consider the relative projections onto $p$ and the participants in $\tilde{q}$.
                    For the recipient $r$,
                    \begin{align}
                        \gtc{G_s} \wrt (p,r) = p {!} r \rtBraces{ \msg i<S_i> \rtc. (\gtc{G_i} \wrt (p,r)) }_{i \in I}.
                        \label{eq:mpstAPCP:outRecvRel}
                    \end{align}
                    For each $q \in \deps$, by \Cref{d:mpstAPCP:hdep}, $\ddep((q,p), \gtc{G}) \neq \tSkip \rtc. \rtc{R}$ for any $\rtc{R}$.
                    That is, since $p$ is the sender of the exchange, for each $q \in \deps$, by the definitions in \Cref{alg:mpstAPCP:relProj},
                    \begin{align}
                        \gtc{G_s} \wrt (p,q) = (p {!} r) {!} q \rtBraces{ i \rtc. (\gtc{G_i} \wrt (p,q)) }_{i \in I}.
                        \label{eq:mpstAPCP:outDepsRel}
                    \end{align}
                    On the other hand, for each $q \in \tilde{q} \setminus \deps \setminus \{r\}$,
                    \begin{align}
                        \gtc{G_s} \wrt (p,q) = \tSkip \rtc. (\gtc{G_{i'}} \wrt (p,q))
                        \label{eq:mpstAPCP:outSkipRel}
                    \end{align}
                    for any $i' \in I$, because for each $i,j \in I$,
                    \begin{align}
                        \gtc{G_i} \wrt (p,q) = \gtc{G_j} \wrt (p,q).
                        \label{eq:mpstAPCP:outSkipSame}
                    \end{align}

                    Let us take stock of the types we expect for each of the router's channels.
                    \begin{align}
                        & \text{For $\ci{\mu}{p}$ we expect}
                        &
                        \ol{\gtc{G_s} \locProjD^{\pri_{\globCtx{C}}} p}
                        &= \ol{\oplus^{\pri_{\globCtx{C}}} \{ i : \msgProp{S_i} \tensor^{\pri_{\globCtx{C}}+1} (\gtc{G_i} \locProjD^{\pri_{\globCtx{C}}+4} p) \}_{i \in I}}
                        \nonumber \\
                        &
                        &
                        &= \&^{\pri_{\globCtx{C}}} \{ i : \ol{\msgProp{S_i}} \parr^{\pri_{\globCtx{C}}+1} \ol{(\gtc{G_i} \locProjD^{\pri_{\globCtx{C}}+4} p)} \}_{i \in I}.
                        \label{eq:mpstAPCP:outCiType}
                        \allowdisplaybreaks \\
                        & \text{For $\crt{p}{r}$ we expect}
                        &
                        \relProp{p}{r}{\pri_{\globCtx{C}}}{\gtc{G_s} \wrt (p,r)}
                        &= \relProp{p}{r}{\pri_{\globCtx{C}}}{p {!} r \rtBraces{ \msg i<S_i> \rtc. (\gtc{G_i} \wrt (p,r)) }_{i \in I}}
                        \tag{cf.\ \eqref{eq:mpstAPCP:outRecvRel}}
                        \nonumber \\
                        &
                        &
                        &= \oplus^{\pri_{\globCtx{C}}+1} \{ i : \msgProp{S_i} \tensor^{\pri_{\globCtx{C}}+2} \relProp{p}{r}{\pri_{\globCtx{C}}+4}{\gtc{G_i} \wrt (p,r)} \}_{i \in I}.
                        \label{eq:mpstAPCP:outCrtRecvType}
                        \allowdisplaybreaks \\
                        & \text{For each $q \in \deps$,}
                        \nonumber \\
                        & \text{for $\crt{p}{q}$ we expect}
                        &
                        \relProp{p}{q}{\pri_{\globCtx{C}}}{\gtc{G_s} \wrt (p,q)}
                        &= \relProp{p}{q}{\pri_{\globCtx{C}}}{(p {!} r) {!} q \rtBraces{ i \rtc. (\gtc{G_i} \wrt (p,q)) }_{i \in I}}
                        \tag{cf.\ \eqref{eq:mpstAPCP:outDepsRel}}
                        \nonumber \\
                        &
                        &
                        &= \oplus^{\pri_{\globCtx{C}}+1} \{ i : \relProp{p}{q}{\pri_{\globCtx{C}}+4}{\gtc{G_i} \wrt (p,q)} \}_{i \in I}.
                        \label{eq:mpstAPCP:outCrtDepsType}
                        \allowdisplaybreaks \\
                        & \text{For each $q \in \tilde{q} \setminus \deps \setminus \{r\}$,}
                        \span
                        \nonumber \\
                        & \text{for $\crt{p}{q}$ we expect}
                        &
                        \relProp{p}{q}{\pri_{\globCtx{C}}}{\gtc{G_s} \wrt (p,q)}
                        &= \relProp{p}{q}{\pri_{\globCtx{C}}}{\tSkip \rtc. (\gtc{G_{i'}} \wrt (p,q))}
                        \tag{cf.\ \eqref{eq:mpstAPCP:outSkipRel}}
                        \nonumber \\
                        &
                        &
                        &= \relProp{p}{q}{\pri_{\globCtx{C}}+4}{\gtc{G_{i'}} \wrt (p,q)}
                        ~\text{for any $i' \in I$.}
                        \label{eq:mpstAPCP:outCrtSkipType}
                    \end{align}

                    Let us now consider the process returned by \Cref{alg:mpstAPCP:router}, with each prefix marked with a number:
                    \[
                        \lmed{p}{\tilde{q}}{\gtc{G_s}} =
                        \underbrace{\pBra* \ci{\mu}{p} > \big\{ i : \big.}_{1}~
                        \underbrace{\pSel* \crt{p}{r} < i}_{2_i} \cdot
                        \underbrace{{(\pSel* \crt{p}{q} < i)}_{q \in \deps}}_{3_i} \cdot
                        \underbrace{\pIn \ci{\mu}{p}(v)}_{4_i} ;
                        \underbrace{\pOut* \crt{p}{r}[w]}_{5_i} \cdot
                        (\pFwd [v<>w] \| \lmed{p}{\tilde{q}}{\gtc{G_i}})
                        \big. \big\}_{i \in I}
                    \]
                    For each $i' \in I$, let $\globCtx{C_{i'}} \deq \globCtx{C}[s {!} r (\gtBraces{ \msg i<S_i> \gtc. \gtc{G_i} }_{i \in I \setminus \{i'\}} \cup \gtBraces{ \msg i'<S_{i'}> \gtc. \globHole })]$.
                    Clearly, $\gtc{G_{i'}} \gSub{\globCtx{C_{i'}}} \gtc{G}$.
                    Also, because we are not adding recursion binders, the current value of $\tilde{q}$ is appropriate for the IH.
                    With this context $\globCtx{C_{i'}}$ and $\tilde{q}$, we apply the IH to obtain the typing of $\lmed{p}{\tilde{q}}{\gtc{G_{i'}}}$, where priorities start at $\ctxPri(\globCtx{C_{i'}}) = \ctxPri(\globCtx{C})+4 = \pri_{\globCtx{C}}+4$ (cf.\ \Cref{d:mpstAPCP:ctxPri}).
                    Following these typings, \Cref{f:mpstAPCP:tder} gives the typing of $\lmed{p}{\tilde{q}}{\gtc{G_s}}$, referring to parts of the process by the number marking its foremost prefix above.

                    Clearly, the priorities in the derivation of \Cref{f:mpstAPCP:tder} meet all requirements.
                    The order of the applications of Rule~\ruleLabel{typ-sel$\star$} for each $q \in \deps$ does not matter, since the selections are asynchronous.

                    \begin{figure}[t!]
                        \ifthesis\else \begin{mdframed} \fi
                            \begin{bussproof}
                                \bussAx[\ruleLabel{typ-fwd}]{
                                    \forall i \in I.~ \vdash \pFwd [v<>w] \typInf v:\ol{\msgProp{S_i}}, w:\msgProp{S_i}
                                }
                                \bussAssume{
                                    \begin{array}[b]{@{}l@{}}
                                        \forall i \in I.~ \vdash \lmed{p}{\tilde{q}}{\gtc{G_i}}
                                        \\
                                        {}\typInf{} \begin{array}[t]{@{}l@{}}
                                            \ci{\mu}{p}:\ol{(\gtc{G_i} \locProjD^{\pri_{\globCtx{C}}+4} p)},
                                            \\
                                            {\big(\crt{p}{q}:\relProp{p}{q}{\pri_{\globCtx{C}}+4}{\gtc{G_i} \wrt (p,q)}\big)}_{q \in \tilde{q}}
                                        \end{array}
                                    \end{array}
                                }
                                \bussBin[\ruleLabel{typ-par}]{
                                    \forall i \in I.~ \vdash \pFwd [v<>w] \| \lmed{p}{\tilde{q}}{\gtc{G_i}} \typInf \begin{array}[t]{@{}l@{}}
                                        \ci{\mu}{p}:\ol{(\gtc{G_i} \locProjD^{\pri_{\globCtx{C}}+4} p)},
                                        v:\ol{\msgProp{S_i}},
                                        w:\msgProp{S_i},
                                        \\
                                        {\big(\crt{p}{q}:\relProp{p}{q}{\pri_{\globCtx{C}}+4}{\gtc{G_i} \wrt (p,q)}\big)}_{q \in \tilde{q}}
                                    \end{array}
                                }
                                \bussUn[\ruleLabel{typ-send$\star$}]{
                                    \forall i \in I.~ \vdash 5_i \typInf \begin{array}[t]{@{}l@{}}
                                        \ci{\mu}{p}:\ol{(\gtc{G_i} \locProjD^{\pri_{\globCtx{C}}+4} p)},
                                        v:\ol{\msgProp{S_i}},
                                        \\
                                        \crt{p}{r}:\msgProp{S_i} \tensor^{\pri_{\globCtx{C}}+2} \relProp{p}{r}{\pri_{\globCtx{C}}+4}{\gtc{G_i} \wrt (p,r)},
                                        \\
                                        {\big(\crt{p}{q}:\relProp{p}{q}{\pri_{\globCtx{C}}+4}{\gtc{G_i} \wrt (p,q)}\big)}_{q \in \tilde{q} \setminus \{r\}}
                                    \end{array}
                                }
                                \bussUn[\ruleLabel{typ-recv}]{
                                    \forall i \in I.~ \vdash 4_i \typInf \begin{array}[t]{@{}l@{}}
                                        \ci{\mu}{p}:\ol{\msgProp{S_i}} \parr^{\pri_{\globCtx{C}}+1} \ol{(\gtc{G_i} \locProjD^{\pri_{\globCtx{C}}+4} p)},
                                        \\
                                        \crt{p}{r}:\msgProp{S_i} \tensor^{\pri_{\globCtx{C}}+2} \relProp{p}{r}{\pri_{\globCtx{C}}+4}{\gtc{G_i} \wrt (p,r)},
                                        \\
                                        {\big(\crt{p}{q}:\relProp{p}{q}{\pri_{\globCtx{C}}+4}{\gtc{G_i} \wrt (p,q)}\big)}_{q \in \tilde{q} \setminus \{r\}}
                                    \end{array}
                                }
                                \bussUn[$\ruleLabel{typ-sel$\star$}^\ast$]{
                                    \forall i \in I.~ \vdash 3_i \typInf \begin{array}[t]{@{}lr@{}}
                                        \ci{\mu}{p}:\ol{\msgProp{S_i}} \parr^{\pri_{\globCtx{C}}+1} \ol{(\gtc{G_i} \locProjD^{\pri_{\globCtx{C}}+4} p)},
                                        \\
                                        \crt{p}{r}:\msgProp{S_i} \tensor^{\pri_{\globCtx{C}}+2} \relProp{p}{r}{\pri_{\globCtx{C}}+4}{\gtc{G_i} \wrt (p,r)},
                                        \\
                                        {\big(\crt{p}{q}:\oplus^{\pri_{\globCtx{C}}+1} \{ i : \relProp{p}{q}{\pri_{\globCtx{C}}+4}{\gtc{G_i} \wrt (p,q)} \}_{i \in I}\big)}_{q \in \deps},
                                        \\
                                        {\big(\crt{p}{q}:\relProp{p}{q}{\pri_{\globCtx{C}}+4}{\gtc{G_i} \wrt (p,q)}\big)}_{q \in \tilde{q} \setminus \deps}
                                    \end{array}
                                }
                                \bussUn[\ruleLabel{typ-sel$\star$}]{
                                    \forall i \in I.~ \vdash 2_i \typInf \begin{array}[t]{@{}lr@{}}
                                        \ci{\mu}{p}:\ol{\msgProp{S_i}} \parr^{\pri_{\globCtx{C}}+1} \ol{(\gtc{G_i} \locProjD^{\pri_{\globCtx{C}}+4} p)},
                                        \\
                                        \crt{p}{r}:\oplus^{\pri_{\globCtx{C}}+1} \{ i : \msgProp{S_i} \tensor^{\pri_{\globCtx{C}}+2} \relProp{p}{r}{\pri_{\globCtx{C}}+4}{\gtc{G_i} \wrt (p,r)} \}_{i \in I},
                                        \\
                                        {\big(\crt{p}{q}:\oplus^{\pri_{\globCtx{C}}+1} \{ i : \relProp{p}{q}{\pri_{\globCtx{C}}+4}{\gtc{G_i} \wrt (p,q)} \}_{i \in I}\big)}_{q \in \deps},
                                        \\
                                        {\big(\crt{p}{q}:\relProp{p}{q}{\pri_{\globCtx{C}}+4}{\gtc{G_i} \wrt (p,q)}\big)}_{q \in \tilde{q} \setminus \deps}
                                        & \text{(cf.\ \eqref{eq:mpstAPCP:outSkipSame})}
                                    \end{array}
                                }
                                \bussUn[\ruleLabel{typ-bra}]{
                                    \vdash \lmed{p}{\tilde{q}}{\gtc{G_s}} = 1 \typInf \begin{array}[t]{@{}lr@{}}
                                        \ci{\mu}{p}:\&^{\pri_{\globCtx{C}}} \{ i : \ol{\msgProp{S_i}} \parr^{\pri_{\globCtx{C}}+1} \ol{(\gtc{G_i} \locProjD^{\pri_{\globCtx{C}}+4} p)} \}_{i \in I},
                                        & \text{(cf.\ \eqref{eq:mpstAPCP:outCiType})}
                                        \\
                                        \crt{p}{r}:\oplus^{\pri_{\globCtx{C}}+1} \{ i : \msgProp{S_i} \tensor^{\pri_{\globCtx{C}}+2} \relProp{p}{r}{\pri_{\globCtx{C}}+4}{\gtc{G_i} \wrt (p,r)} \}_{i \in I},
                                        & \text{(cf.\ \eqref{eq:mpstAPCP:outCrtRecvType})}
                                        \\
                                        {\big(\crt{p}{q}:\oplus^{\pri_{\globCtx{C}}+1} \{ i : \relProp{p}{q}{\pri_{\globCtx{C}}+4}{\gtc{G_i} \wrt (p,q)} \}_{i \in I}\big)}_{q \in \deps},
                                        & \text{(cf.\ \eqref{eq:mpstAPCP:outCrtDepsType})}
                                        \\
                                        {\big(\crt{p}{q}:\relProp{p}{q}{\pri_{\globCtx{C}}+4}{\gtc{G_{i'}} \wrt (p,q)}\big)}_{q \in \tilde{q} \setminus \deps}
                                        & \text{(cf.\ \eqref{eq:mpstAPCP:outCrtSkipType})}
                                    \end{array}
                                }
                            \end{bussproof}
                        \ifthesis\else \end{mdframed} \fi

                        \caption{Typing derivation used in the proof of \Cref{t:mpstAPCP:routerTypes}.}
                        \label{f:mpstAPCP:tder}
                    \end{figure}

                \item
                    If $p = r$, then $p$ is the recipient (\cref{li:mpstAPCP:rtrRecv}).
                    This case is analogous to the previous one.

                \item
                    If $p \notin \{r,s\}$ (\cref{li:mpstAPCP:rtrElse}), then further analysis depends on whether the exchange is a dependency for $p$.
                    Let
                    \begin{align*}
                        \dep_s
                        &\deq (s \in \tilde{q} \wedge \hdep(p, s, \gtc{G}))
                        &
                        & \text{(as on \cref{li:mpstAPCP:rtrDepS}), and}
                        \\
                        \dep_r
                        &\deq (r \in \tilde{q} \wedge \hdep(p, r, \gtc{G}))
                        &
                        & \text{(as on \cref{li:mpstAPCP:rtrDepR}).}
                    \end{align*}
                    To see what truth of $\dep_s$ and $\dep_r$ means, we follow \Cref{d:mpstAPCP:hdep} and the definitions in \Cref{alg:mpstAPCP:relProj}.
                    \begin{align}
                        G_s \wrt (p,s)
                        &= \begin{cases}
                            (s {!} r) {!} p \rtBraces{ i \rtc. (\gtc{G_i} \wrt (p,s)) }_{i \in I}
                            & \text{if $\dep_s$ is true}
                            \\
                            \tSkip \rtc. (\gtc{G_{i'}} \wrt (p,s)) ~\text{for any $i' \in I$}
                            & \text{otherwise}
                        \end{cases}
                        \label{eq:mpstAPCP:rDepSRel}
                        \\
                        G_s \wrt (p,r)
                        &= \begin{cases}
                            (r {?} s) {!} p \rtBraces{ i \rtc. (\gtc{G_i} \wrt (p,r)) }_{i \in I}
                            & \text{if $\dep_r$ is true}
                            \\
                            \tSkip \rtc. (\gtc{G_{i'}} \wrt (p,r)) ~\text{for any $i' \in I$}
                            & \text{otherwise}
                        \end{cases}
                        \label{eq:mpstAPCP:rDepRRel}
                    \end{align}
                    Let us also consider the relative projections onto $p$ and the participants in $\tilde{q}$ besides $r$ and $s$, which follow by the relative well-formedness of $\gtc{G_s}$.
                    For each $q \in \tilde{q} \setminus \{r,s\}$,
                    \begin{align}
                        \gtc{G_s} \wrt (p,q) = \tSkip \rtc. (\gtc{G_{i'}} \wrt (p,q))
                        \label{eq:mpstAPCP:rDepOtherRel}
                    \end{align}
                    for any $i' \in I$.

                    The rest of the analysis depends on the truth of $\dep_s$ and $\dep_r$.
                    There are four cases.
                    \begin{enumerate}
                        \item
                            If $\dep_s$ is true and $\dep_r$ is false (\cref{li:mpstAPCP:rtrOnlyS}), let us take stock of the types we expect for each of the router's channels.
                            \begin{align}
                                & \text{For $\ci{\mu}{p}$ we expect}
                                &
                                \ol{\gtc{G_s} \locProjD^{\pri_{\globCtx{C}}} p}
                                &= \ol{\&^{\pri_{\globCtx{C}}+2} \{ i : (\gtc{G_i} \locProjD^{\pri_{\globCtx{C}}+4} p) \}_{i \in I}}
                                \nonumber \\
                                &
                                &
                                &= \oplus^{\pri_{\globCtx{C}}+2} \{ i : \ol{(\gtc{G_i} \locProjD^{\pri_{\globCtx{C}}+4} p)} \}_{i \in I}.
                                \label{eq:mpstAPCP:rDepSCiType}
                                \allowdisplaybreaks \\
                                & \text{For $\crt{p}{s}$ we expect}
                                &
                                \relProp{p}{s}{\pri_{\globCtx{C}}}{\gtc{G_s} \wrt (p,s)}
                                &= \relProp{p}{s}{\pri_{\globCtx{C}}}{(s {!} r) {!} p \rtBraces{ i \rtc. (\gtc{G_i} \wrt (p,s)) }_{i \in I}}
                                \tag{cf.\ \eqref{eq:mpstAPCP:rDepSRel}}
                                \nonumber \\
                                &
                                &
                                &= \&^{\pri_{\globCtx{C}}+1} \{ i : \relProp{p}{s}{\pri_{\globCtx{C}}+4}{\gtc{G_i} \wrt (p,s)} \}_{i \in I}.
                                \label{eq:mpstAPCP:rDepSCrtSType}
                                \allowdisplaybreaks \\
                                & \text{For each $q \in \tilde{q} \setminus \{s\}$,}
                                \nonumber \\
                                & \text{for $\crt{p}{q}$ we expect}
                                & \relProp{p}{q}{\pri_{\globCtx{C}}}{\gtc{G_s} \wrt (p,q)}
                                &= \relProp{p}{q}{\pri_{\globCtx{C}}}{\tSkip \rtc. (\gtc{G_{i'}} \wrt (p,q))}
                                \tag{cf.\ \eqref{eq:mpstAPCP:rDepRRel} and~\eqref{eq:mpstAPCP:rDepOtherRel}}
                                \nonumber \\
                                &
                                &
                                &= \relProp{p}{q}{\pri_{\globCtx{C}}+4}{\gtc{G_{i'}} \wrt (p,q)} ~\text{for any $i' \in I$}.
                                \label{eq:mpstAPCP:rDepSCrtOtherType}
                            \end{align}

                            Similar to case (1), we apply the IH to obtain the typing of $\lmed{p}{\tilde{q}}{\gtc{G_i}}$ for each $i \in I$, starting at priority $\pri_{\globCtx{C}}+4$.
                            We derive the typing of $\lmed{p}{\tilde{q}}{\gtc{G_s}}$:
                            \[
                                \begin{bussproof}
                                    \bussAssume{
                                        \forall i \in I.~ \vdash \lmed{p}{\tilde{q}}{\gtc{G_i}} \typInf \ci{\mu}{p}:\ol{\gtc{G_i} \locProjD^{\pri_{\globCtx{C}}+4} p}, {\big(\crt{p}{q}:\relProp{p}{q}{\pri_{\globCtx{C}}+4}{\gtc{G_i} \wrt (p,q)}\big)}_{q \in \tilde{q}}
                                    }
                                    \bussUn[\ruleLabel{typ-sel$\star$}]{
                                        \forall i \in I.~ \vdash \pSel* \ci{\mu}{p} < i \cdot \lmed{p}{\tilde{q}}{\gtc{G_i}} \typInf \begin{array}[t]{@{}l@{}}
                                            \ci{\mu}{p}:\oplus^{\pri_{\globCtx{C}}+2} \{ i : \ol{\gtc{G_i} \locProjD^{\pri_{\globCtx{C}}+4} p} \}_{i \in I},
                                            \\
                                            {\big(\crt{p}{q}:\relProp{p}{q}{\pri_{\globCtx{C}}+4}{\gtc{G_i} \wrt (p,q)}\big)}_{q \in \tilde{q}}
                                        \end{array}
                                    }
                                    \bussUn[\ruleLabel{typ-bra}]{
                                        \begin{array}{@{}l@{}}
                                            \vdash \lmed{p}{\tilde{q}}{\gtc{G_s}} = \pBra* \crt{p}{s} > {\{ i : \pSel* \ci{\mu}{p} < i \cdot \lmed{p}{\tilde{q}}{\gtc{G_i}} \}_{i \in I}}
                                            \\
                                            {}\typInf{} \begin{array}[t]{@{}lr@{}}
                                                \ci{\mu}{p}:\oplus^{\pri_{\globCtx{C}}+2} \{ i : \ol{\gtc{G_i} \locProjD^{\pri_{\globCtx{C}}+4} p} \}_{i \in I},
                                                & \text{(cf.\ \eqref{eq:mpstAPCP:rDepSCiType})}
                                                \\
                                                \crt{p}{s}:\&^{\pri_{\globCtx{C}}+1} \{ i : \relProp{p}{q}{\pri_{\globCtx{C}}+4}{\gtc{G_i} \wrt (p,q)} \}_{i \in I},
                                                & \text{(cf.\ \eqref{eq:mpstAPCP:rDepSCrtSType})}
                                                \\
                                                {\big(\crt{p}{q}:\relProp{p}{q}{\pri_{\globCtx{C}}+4}{\gtc{G_i'} \wrt (p,q)}\big)}_{q \in \tilde{q}}
                                                & \text{(cf.\ \eqref{eq:mpstAPCP:rDepSCrtOtherType})}
                                            \end{array}
                                        \end{array}
                                    }
                                \end{bussproof}
                            \]

                        \item
                            The case where $\dep_s$ is false and $\dep_r$ is true (\cref{li:mpstAPCP:rtrOnlyR}) is analogous to the previous one.

                        \item
                            If both $\dep_s$ and $\dep_r$ are true (\cref{li:mpstAPCP:rtrBoth} and \eqref{eq:mpstAPCP:patched}), let us once again  take stock of  the types we expect for each of the router's channels.
                            \begin{align}
                                & \text{For $\ci{\mu}{p}$ we expect}
                                &
                                \ol{\gtc{G_s} \locProjD^{\pri_{\globCtx{C}}} p}
                                &= \ol{\&^{\pri_{\globCtx{C}}+2} \{ i : (\gtc{G_i} \locProjD^{\pri_{\globCtx{C}}+4} p) \}_{i \in I}}
                                \nonumber
                                \\
                                &
                                &
                                &= \oplus^{\pri_{\globCtx{C}}+2} \{ i : \ol{(\gtc{G_i} \locProjD^{\pri_{\globCtx{C}}+4} p)} \}_{i \in I}
                                \label{eq:mpstAPCP:rDepCiType}
                                \allowdisplaybreaks \\
                                & \text{For $\crt{p}{s}$ we expect}
                                &
                                \relProp{p}{s}{\pri_{\globCtx{C}}}{\gtc{G_s} \wrt (p,s)}
                                &= \relProp{p}{s}{\pri_{\globCtx{C}}}{(s {!} r) {!} p \rtBraces{ i \rtc. (\gtc{G_i} \wrt (p,s)) }_{i \in I}}
                                \tag{cf.\ \eqref{eq:mpstAPCP:rDepSRel}}
                                \nonumber \\
                                &
                                &
                                &= \&^{\pri_{\globCtx{C}}+1} \{ i : \relProp{p}{s}{\pri_{\globCtx{C}}+4}{\gtc{G_i} \wrt (p,s)} \}_{i \in I}
                                \label{eq:mpstAPCP:rDepCrtSType}
                                \allowdisplaybreaks \\
                                & \text{For $\crt{p}{r}$ we expect}
                                &
                                \relProp{p}{r}{\pri_{\globCtx{C}}}{\gtc{G_s} \wrt (p,r)}
                                &= \relProp{p}{r}{\pri_{\globCtx{C}}}{(r {?} s) {!} p \rtBraces{ i \rtc. (\gtc{G_i} \wrt (p,r)) }_{i \in I}}
                                \tag{cf.\ \eqref{eq:mpstAPCP:rDepRRel}}
                                \nonumber \\
                                &
                                &
                                &= \&^{\pri_{\globCtx{C}}+2} \{ i : \relProp{p}{r}{\pri_{\globCtx{C}}+4}{\gtc{G_i} \wrt (p,r)} \}_{i \in I}
                                \label{eq:mpstAPCP:rDepCrtRType}
                                \allowdisplaybreaks \\
                                & \text{For each $q \in \tilde{q} \setminus \{s,r\}$,}
                                \nonumber \\
                                & \text{for $\crt{p}{q}$ we expect}
                                &
                                \relProp{p}{q}{\pri_{\globCtx{C}}}{\gtc{G_s} \wrt (p,q)}
                                &= \relProp{p}{q}{\pri_{\globCtx{C}}}{\tSkip \rtc. (\gtc{G_{i'}} \wrt (p,q))}
                                \tag{cf.\ \eqref{eq:mpstAPCP:rDepOtherRel}}
                                \nonumber \\
                                &
                                &
                                &= \relProp{p}{q}{\pri_{\globCtx{C}}+4}{\gtc{G_{i'}} \wrt (p,q)} ~\text{for any $i' \in I$}
                                \label{eq:mpstAPCP:rDepCrtOtherType}
                            \end{align}

                            It is clear from~\eqref{eq:mpstAPCP:rDepCrtSType} and~\eqref{eq:mpstAPCP:rDepCrtRType} that the router will receive label $i \in I$ first on~$\crt{p}{s}$ and then $i' \in I$ on $\crt{p}{r}$.
                            We rely on alarm processes (\Cref{d:mpstAPCP:alarm}) to handle the case $i' \neq i$.

                            Similar to case~(1), we apply the IH to obtain the typing of $\lmed{p}{\tilde{q}}{\gtc{G_i}}$ for each $i \in I$, starting at priority $\pri_{\globCtx{C}}+4$.
                            \Cref{f:mpstAPCP:tderii} gives the typing of $\lmed{p}{\tilde{q}}{\gtc{G_s}}$.

                            \begin{figure}[t!]
                                \ifthesis\else \begin{mdframed} \fi
                                    \[
                                        \begin{bussproof}
                                            \bussAssume{
                                                \begin{array}[b]{@{}l@{}}
                                                    \forall i \in I.~
                                                    \vdash \lmed{p}{\tilde{q}}{\gtc{G_i}}
                                                    \\
                                                    \typInf \begin{array}[t]{@{}l@{}}
                                                        \ci{\mu}{p}:\ol{\gtc{G_i} \locProjD^{\pri_{\globCtx{C}}+4} p},
                                                        \\
                                                        {\big(\crt{p}{q}:\relProp{p}{q}{\pri_{\globCtx{C}}+4}{\gtc{G_i} \wrt (p,q)}\big)}_{q \in \tilde{q}}
                                                    \end{array}
                                                \end{array}
                                            }
                                            \bussAx[\ruleLabel{typ-alarm}]{
                                                \begin{array}[b]{@{}l@{}}
                                                    \forall i \in I.~
                                                    \forall i' \in I \setminus \{i\}.~ \vdash \pAlarm(\sff{chs})
                                                    \\
                                                    \typInf \begin{array}[t]{@{}l@{}}
                                                        \ci{\mu}{p}:\ol{\gtc{G_i} \locProjD^{\pri_{\globCtx{C}}+4} p},
                                                        \\
                                                        \crt{p}{s}:\relProp{p}{s}{\pri_{\globCtx{C}}+4}{\gtc{G_i} \wrt (p,s)},
                                                        \\
                                                        \crt{p}{r}:\relProp{p}{r}{\pri_{\globCtx{C}}+4}{\gtc{G_{i'}} \wrt (p,r)},
                                                        \\
                                                        {\big(\crt{p}{q}:\relProp{p}{q}{\pri_{\globCtx{C}}+4}{\gtc{G_i} \wrt (p,q)}\big)}_{q \in \tilde{q} \setminus \{s,r\}}
                                                    \end{array}
                                                \end{array}
                                            }
                                            \bussBin[\ruleLabel{typ-bra}]{
                                                \forall i \in I.~ \begin{array}[t]{@{}l@{}}
                                                    \vdash \pBra* \crt{p}{r} > {\{ i : \lmed{p}{\tilde{q}}{\gtc{G_i}} \} \cup \{ i' : \pAlarm(\sff{chs}) \}_{i' \in I \setminus \{i\}}}
                                                    \\
                                                    \typInf \begin{array}[t]{@{}l@{}}
                                                        \ci{\mu}{p}:\ol{\gtc{G_i} \locProjD^{\pri_{\globCtx{C}}+4} p},
                                                        \crt{p}{s}:\relProp{p}{s}{\pri_{\globCtx{C}}+4}{G_i \wrt (p,s)},
                                                        \\
                                                        \crt{p}{r}:\&^{\pri_{\globCtx{C}}+2} \{ i : \relProp{p}{r}{\pri_{\globCtx{C}}+4}{\gtc{G_{i'}} \wrt (p,r)} \}_{i \in I},
                                                        \\
                                                        {\big(\crt{p}{q}:\relProp{p}{q}{\pri_{\globCtx{C}}+4}{\gtc{G_i} \wrt (p,q)}\big)}_{q \in \tilde{q} \setminus \{s,r\}}
                                                    \end{array}
                                                \end{array}
                                            }
                                            \bussUn[\ruleLabel{typ-sel$\star$}]{
                                                \forall i \in I.~ \begin{array}[t]{@{}l@{}}
                                                    \vdash \pSel* \ci{\mu}{p} < i \cdot \pBra* \crt{p}{r} > {\{ i : \lmed{p}{\tilde{q}}{\gtc{G_i}} \} \cup \{ i' : \pAlarm(\sff{chs}) \}_{i' \in I \setminus \{i\}}}
                                                    \\
                                                    \typInf \begin{array}[t]{@{}l@{}}
                                                        \ci{\mu}{p}:\oplus^{\pri_{\globCtx{C}}+2} \{ i : \ol{\gtc{G_i} \locProjD^{\pri_{\globCtx{C}}+4} p} \}_{i \in I},
                                                        \crt{p}{s}:\relProp{p}{s}{\pri_{\globCtx{C}}+4}{\gtc{G_i} \wrt (p,s)},
                                                        \\
                                                        \crt{p}{r}:\&^{\pri_{\globCtx{C}}+2} \{ i : \relProp{p}{r}{\pri_{\globCtx{C}}+4}{\gtc{G_{i'}} \wrt (p,r)} \}_{i \in I},
                                                        \\
                                                        {\big(\crt{p}{q}:\relProp{p}{q}{\pri_{\globCtx{C}}+4}{\gtc{G_i} \wrt (p,q)}\big)}_{q \in \tilde{q} \setminus \{s,r\}}
                                                    \end{array}
                                                \end{array}
                                            }
                                            \bussUn[\ruleLabel{typ-bra}]{
                                                \begin{array}{@{}l@{}}
                                                    \vdash \lmed{p}{\tilde{q}}{\gtc{G_s}} = \pBra* \crt{p}{s} > {\{ i : \pSel* \ci{\mu}{p} < i \cdot \pBra* \crt{p}{r} > {\{ i : \lmed{p}{\tilde{q}}{\gtc{G_i}} \} \cup \{ i' : \pAlarm(\sff{chs}) \}_{i' \in I \setminus \{i\}}} \}_{i \in I}}
                                                    \\
                                                    \typInf \begin{array}[t]{@{}lr@{}}
                                                        \ci{\mu}{p}:\oplus^{\pri_{\globCtx{C}}+2} \{ i : \ol{\gtc{G_i} \locProjD^{\pri_{\globCtx{C}}+4} p} \}_{i \in I},
                                                        & \text{(cf.\ \eqref{eq:mpstAPCP:rDepCiType})}
                                                        \\
                                                        \crt{p}{s}:\&^{\pri_{\globCtx{C}}+1} \{ i : \relProp{p}{s}{\pri_{\globCtx{C}}+4}{\gtc{G_i} \wrt (p,s)} \}_{i \in I},
                                                        & \text{(cf.\ \eqref{eq:mpstAPCP:rDepCrtSType})}
                                                        \\
                                                        \crt{p}{r}:\&^{\pri_{\globCtx{C}}+2} \{ i : \relProp{p}{r}{\pri_{\globCtx{C}}+4}{\gtc{G_{i'}} \wrt (p,r)} \}_{i \in I},
                                                        & \text{(cf.\ \eqref{eq:mpstAPCP:rDepCrtRType})}
                                                        \\
                                                        {\big(\crt{p}{q}:\relProp{p}{q}{\pri_{\globCtx{C}}+4}{\gtc{G_{i'}} \wrt (p,q)}\big)}_{q \in \tilde{q} \setminus \{s,r\}}
                                                        & \text{(cf.\ \eqref{eq:mpstAPCP:rDepCrtOtherType})}
                                                    \end{array}
                                                \end{array}
                                            }
                                        \end{bussproof}
                                    \]
                                \ifthesis\else \end{mdframed} \fi

                                \caption{Typing derivation used in the proof of \Cref{t:mpstAPCP:routerTypes}, where $\sff{chs} = \{\ci{\mu}{p}\} \cup \{\crt{p}{q} \mid q \in \tilde{q}\}$.}
                                \label{f:mpstAPCP:tderii}
                            \end{figure}

                        \item
                            If both $\dep_s$ and $\dep_r$ are false, let us again  take stock of  the types we expect for each of the router's channels.
                            \begin{align*}
                                & \text{For $\ci{\mu}{p}$ we expect}
                                & \ol{\gtc{G_s} \locProjD^{\pri_{\globCtx{C}}} p}
                                &= \ol{\gtc{G_{i'}} \locProjD^{\pri_{\globCtx{C}}+4} p} ~\text{for any $i' \in I$}.
                                \\
                                & \text{For each $q \in \tilde{q}$,}
                                \\
                                & \text{for $\crt{p}{q}$ we expect}
                                & \relProp{p}{q}{\pri_{\globCtx{C}}}{\gtc{G_s} \wrt (p,q)}
                                &= \relProp{p}{q}{\pri_{\globCtx{C}}}{\tSkip \rtc. (\gtc{G_{i'}} \wrt (p,q))}
                                \tag{cf.\ \eqref{eq:mpstAPCP:rDepSRel},~\eqref{eq:mpstAPCP:rDepRRel} and~\eqref{eq:mpstAPCP:rDepOtherRel}}
                                \\
                                &
                                &
                                &= \relProp{p}{q}{\pri_{\globCtx{C}}+4}{\gtc{G_{i'}} \wrt (p,q)} ~\text{for any $i' \in I$.}
                            \end{align*}
                            Similar to case (1), we apply the IH to obtain the typing of $\lmed{p}{\tilde{q}}{\gtc{G_{i'}}}$, starting at priority $\pri_{\globCtx{C}}+4$.
                            This directly proves the thesis.
                    \end{enumerate}
            \end{enumerate}

        \item
            \emph{Recursive definition}: $\gtc{G_s} = \mu Z \gtc. \gtc{G'}$ (\cref{li:mpstAPCP:rtrRecDef}).

            Let
            \begin{align}
                \tilde{q}' \deq \{ q \in \tilde{q} \mid \gtc{G_s} \wrt (p,q) \neq \tEnd \}
                \label{eq:mpstAPCP:rRecDefqtildeprime}
            \end{align}
            (as on \cref{li:mpstAPCP:rtrRecComp}).
            We consider the relative projections onto $p$ and the participants in $\tilde{q}$.
            For each $q \in \tilde{q}'$, we know $\gtc{G_s} \wrt (p,q) \neq \tEnd$, while for each $q \in \tilde{q} \setminus \tilde{q}'$, we know $\gtc{G_s} \wrt (p,q) = \tEnd$.
            More precisely, by \Cref{d:mpstAPCP:relProj}, for each $q \in \tilde{q}'$,
            \begin{align}
                \gtc{G_s} \wrt (p,q) = (\mu Z \gtc. \gtc{G'}) \wrt (p,q) = \mu Z \rtc. (\gtc{G'} \wrt (p,q)).
                \label{eq:mpstAPCP:rDefRelPrime}
            \end{align}
            and thus
            \[
                \gtc{G'} \wrt (p,q) \neq \tSkip^\ast \rtc. \tEnd ~\text{and}~ \gtc{G'} \wrt (p,q) \neq \tSkip^\ast \rtc. Z.
            \]
            For each $q \in \tilde{q} \setminus \tilde{q}'$,
            \begin{align}
                \gtc{G_s} \wrt (p,q) = (\mu Z \gtc. \gtc{G'}) \wrt (p,q) = \tEnd,
                \label{eq:mpstAPCP:rDefRelOth}
            \end{align}
            and thus
            \[
                \gtc{G'} \wrt (p,q) = \tSkip^\ast \rtc. \tEnd ~\text{or}~ \gtc{G'} \wrt (p,q) = \tSkip^\ast \rtc. Z.
            \]

            Further analysis depends on whether $\tilde{q}' = \emptyset$ or not. We thus examine two cases:
            \begin{itemize}
                \item
                    If $\tilde{q}' = \emptyset$ (\cref{li:mpstAPCP:rtrRecRet}), let us consider the local projection $\gtc{G_s} \locProjD^{\pri_{\globCtx{C}}} p$.
                    We prove that $\gtc{G_s} \locProjD^{\pri_{\globCtx{C}}} p = \bullet$.

                    Suppose, for contradiction, that $\gtc{G_s} \locProjD^{\pri_{\globCtx{C}}} p \neq \bullet$.
                    Then, by the definitions in \Cref{f:mpstAPCP:locProj}, $\gtc{G'} \locProjD^{\pri_{\globCtx{C}}} p \neq X$ and $\gtc{G'} \locProjD^{\pri_{\globCtx{C}}} p \neq \bullet$.
                    That is, $\gtc{G'} \locProjD^{\pri_{\globCtx{C}}} p$ contains communications or some recursion variable other than $Z$.
                    However, communications in $\gtc{G'} \locProjD^{\pri_{\globCtx{C}}} p$ originate from exchanges in $\gtc{G'}$, either involving $p$ and some $q \in \tilde{q}$, or as a dependency on an exchange involving some $q \in \tilde{q}$.
                    Moreover, recursion variables in $\gtc{G'} \locProjD^{\pri_{\globCtx{C}}} p$ originate from recursion variables in $\gtc{G'}$.
                    But this would mean that for this $q$, $\gtc{G'} \wrt (p,q)$ contains interactions or recursion variables, contradicting \eqref{eq:mpstAPCP:rDefRelOth}.
                    Therefore, it cannot be the case that $\gtc{G_s} \locProjD^{\pri_{\globCtx{C}}} p \neq \bullet$.

                    Let us take stock of the types we expect for each of the router's channels.
                    For now, we omit the substitutions in the types.
                    \begin{align*}
                        & \text{For $\ci{\mu}{p}$ we expect}
                        &
                        \ol{\gtc{G_s} \locProjD^{\pri_{\globCtx{C}}} p}
                        &= \ol{\bullet} = \bullet.
                        \\
                        & \text{For each $q \in \tilde{q}$, for $\crt{p}{q}$ we expect}
                        &
                        \relProp{p}{q}{\pri_{\globCtx{C}}}{\gtc{G_s} \wrt (p,q)}
                        &= \relProp{p}{q}{\pri_{\globCtx{C}}}{\tEnd} = \bullet.
                        \tag{cf.\ \eqref{eq:mpstAPCP:rDefRelOth}}
                    \end{align*}
                    Because all expected types are $\bullet$, the substitutions do not affect the types, so we can omit them altogether.

                    First we apply Rule~\ruleLabel{typ-inact}, giving us an arbitrary recursive context, and thus the recursive context we need.
                    Then, we apply Rule~\ruleLabel{typ-end} for $\ci{\mu}{p}$ and for $\crt{p}{q}$ for each $q \in \tilde{q}$, and obtain the typing of $\lmed{p}{\tilde{q}}{\gtc{G_s}}$ (omitting the recursive context):
                    \[
                        \vdash \lmed{p}{\tilde{q}}{\gtc{G_s}} = \0 \typInf \ci{\mu}{p}:\bullet, {(\crt{p}{q}:\bullet)}_{q \in \tilde{q}}
                    \]

                \item
                    If $\tilde{q}' \neq \emptyset$ (\cref{li:mpstAPCP:rtrRecInact}), then, following similar reasoning as in the previous case, ${\gtc{G_s} \locProjD^{\pri_{\globCtx{C}}} p = \mu Z . (\gtc{G'} \locProjD^{\pri_{\globCtx{C}}} p)}$.
                    We take stock of the types we expect for each of the router's channels.
                    Note that, because of the recursive definition on $Z$ in $\gtc{G_s}$, there cannot be another recursive definition in the context $\globCtx{C}$ capturing the recursion variable $Z$.
                    Therefore, by \Cref{d:mpstAPCP:binders}, $Z \notin \widetilde{X_{\globCtx{C}}}$.
                    \begin{align}
                        & \text{For $\ci{\mu}{p}$ we expect}
                        &
                        &\deepUnfold(\ol{\gtc{G_s} \locProjD^{\pri_{\globCtx{C}}} p},\ldots)
                        \nonumber \\
                        &
                        &
                        &{} = \deepUnfold(\ol{\mu Z . (\gtc{G'} \locProjD^{\pri_{\globCtx{C}}} p)},\ldots)
                        \nonumber \\
                        &
                        &
                        &{} = \deepUnfold(\mu Z . \ol{\gtc{G'} \locProjD^{\pri_{\globCtx{C}}} p},\ldots)
                        \nonumber \\
                        &
                        &
                        &{} = \mu Z . \deepUnfold(\ol{\gtc{G'} \locProjD^{\pri_{\globCtx{C}}} p}, {(X,t_X,\ol{\gtc{G}_X \locProjD^{\pri_X} p})}_{X \in \widetilde{X_{\globCtx{C}}}}).
                        \label{eq:mpstAPCP:rRecDefMu}
                        \allowdisplaybreaks \\
                        & \text{For each $q \in \tilde{q}'$,}
                        \nonumber \\
                        & \text{for $\crt{p}{q}$ we expect}
                        &
                        &\deepUnfold(\relProp{p}{q}{\pri_{\globCtx{C}}}{\gtc{G_s} \wrt (p,q)},\ldots)
                        \nonumber \\
                        &
                        &
                        &{} = \deepUnfold(\relProp{p}{q}{\pri_{\globCtx{C}}}{\mu Z . (\gtc{G'} \wrt (p,q))},\ldots)
                        \nonumber \\
                        &
                        &
                        &{} = \deepUnfold(\mu Z . \relProp{p}{q}{\pri_{\globCtx{C}}}{\gtc{G'} \wrt (p,q)},\ldots)
                        \tag{cf.\ \eqref{eq:mpstAPCP:rDefRelPrime}}
                        \nonumber \\
                        &
                        &
                        &= \mu Z . \deepUnfold( \begin{array}[t]{@{}l@{}}
                            \relProp{p}{q}{\pri_{\globCtx{C}}}{\gtc{G'} \wrt (p,q)}, \\
                            {(X,t_X,\relProp{p}{q}{\pri_X}{\gtc{G}_X \wrt (p,q)})}_{X \in \widetilde{X_{\globCtx{C}}}}).
                        \end{array}
                        \label{eq:mpstAPCP:rRecDefqprime}
                        \allowdisplaybreaks \\
                        & \text{For each $q \in \tilde{q} \setminus \tilde{q}'$,}
                        \nonumber \\
                        & \text{for $\crt{p}{q}$ we expect}
                        &
                        &\deepUnfold(\relProp{p}{q}{\pri_{\globCtx{C}}}{\gtc{G_s} \wrt (p,q)},\ldots)
                        \nonumber \\
                        &
                        &
                        &{} = \deepUnfold(\relProp{p}{q}{\pri_{\globCtx{C}}}{\tEnd},\ldots)
                        \nonumber \\
                        &
                        &
                        &{} = \deepUnfold(\bullet,\ldots)
                        \tag{cf.\ \eqref{eq:mpstAPCP:rDefRelOth}}
                        \nonumber \\
                        &
                        &
                        &{} = \bullet.
                        \label{eq:mpstAPCP:rRecDefOthers}
                    \end{align}
                    We also need an assignment in the recursive context for every $X \in \widetilde{X_{\globCtx{C}}}$, but not for~$Z$.

                    Let $\globCtx{C'} = \globCtx{C}[\mu Z \gtc. \globHole]$.
                    Clearly, $\gtc{G'} \gSub{\globCtx{C'}} \gtc{G}$.
                    Let us first establish some facts about the recursion binders, priorities, and active participants related to $\globCtx{C'}$, $\gtc{G'}$, and $Z$:
                    \begin{itemize}
                        \item
                            $\widetilde{X_{\globCtx{C'}}} = \ctxBind(\globCtx{C'}) = (\ctxBind(\globCtx{C}), Z) = (\widetilde{X_{\globCtx{C}}}, Z)$ (cf.\ \Cref{d:mpstAPCP:binders}).
                        \item
                            $\gtc{G}_Z = \recDef(Z,\gtc{G}) = \gtc{G'}$, as proven by the context $\globCtx{C'}$ (cf.\ \Cref{d:mpstAPCP:recDef}).
                        \item
                            $\widetilde{Y_Z} = \subBind(\mu Z \gtc. \gtc{G}_Z,\gtc{G}) = \ctxBind(\globCtx{C}) = \widetilde{X_{\globCtx{C}}}$.
                        \item
                            $\pri_{\globCtx{C'}} = \ctxPri(\globCtx{C'}) = \ctxPri(\globCtx{C}) = \pri_{\globCtx{C}}$, and $\pri_Z = \varPri(Z,\gtc{G}) = \ctxPri(\globCtx{C}) = \pri_{\globCtx{C}}$, and hence $\pri_{\globCtx{C'}} = \pri_Z$ (cf.\ \Cref{d:mpstAPCP:ctxPri}).
                        \item
                            $\tilde{q}_Z = \tilde{q}'$ (cf.\ \Cref{d:mpstAPCP:activeParts} and \eqref{eq:mpstAPCP:rRecDefqtildeprime}).
                    \end{itemize}

                    Because $\widetilde{X_{\globCtx{C'}}} = (\widetilde{X_{\globCtx{C}}}, Z)$ and $\tilde{q}' = \tilde{q}_Z$, $\tilde{q}'$ is appropriate for the IH.
                    We apply the IH on $\globCtx{C'}$, $\gtc{G'}$, and $\tilde{q}'$ to obtain a typing for $\lmed{p}{\tilde{q}'}{\gtc{G'}}$, where we immediately make use of the facts established above.
                    We give the assignment to $Z$ in the recursive context separate from those for the recursion variables in $\widetilde{X_{\globCtx{C}}}$.
                    Also, by \Cref{p:mpstAPCP:deepUnfold}, we can write the final unfolding on $Z$ in the types separately.
                    For example, the type for $\ci{\mu}{p}$ is
                    \begin{align*}
                        & \deepUnfold(\ol{\gtc{G'} \locProjD^{\pri_{\globCtx{C'}}} p}, {(X,t_X,\ol{\gtc{G}_X \locProjD^{\pri_X} p})}_{X \in \widetilde{X_{\globCtx{C'}}}})
                        \\
                        &= \deepUnfold(\ol{\gtc{G'} \locProjD^{\pri_{\globCtx{C}}} p}, {(X,t_X,\ol{\gtc{G}_X \locProjD^{\pri_X} p})}_{X \in (\widetilde{X_{\globCtx{C}}}, Z)})
                        \\
                        &= \deepUnfold(\ol{\gtc{G'} \locProjD^{\pri_{\globCtx{C}}} p}, \big({(X,t_X,\ol{\gtc{G}_X \locProjD^{\pri_X} p})}_{X \in \widetilde{X_{\globCtx{C}}}}, (Z,t_Z,\ol{\gtc{G}_Z \locProjD^{\pri_Z} p})\big))
                        \\
                        &= \deepUnfold(\ol{\gtc{G'} \locProjD^{\pri_{\globCtx{C}}} p}, \big({(X,t_X,\ol{\gtc{G}_X \locProjD^{\pri_X} p})}_{X \in \widetilde{X_{\globCtx{C}}}}, (Z,t_Z,\ol{\gtc{G'} \locProjD^{\pri_{\globCtx{C}}} p})\big))
                        \\
                        &= \unfold^{t_Z}\big( \mu Z . \deepUnfold(\ol{\gtc{G'} \locProjD^{\pri_{\globCtx{C}}} p}, {(X,t_X,\ol{\gtc{G}_X \locProjD^{\pri_X} p})}_{X \in \widetilde{X_{\globCtx{C}}}}) \big).
                    \end{align*}

                    The resulting typing is as follows:
                    \begin{align*}
                        &{
                            \left(
                                X:\left(
                                    \begin{array}{@{}l@{}}
                                        \deepUnfold(\ol{\gtc{G}_X \locProjD^{\pri_X} p}, {(Y,t_Y,\ol{\gtc{G}_Y \locProjD^{\pri_Y} p})}_{Y \in \widetilde{Y_X}}),
                                        \\[4pt]
                                        {
                                            \Big(
                                                \deepUnfold(\relProp{p}{q}{\pri_X}{\gtc{G}_X \wrt (p,q)}, {(Y,t_Y,\relProp{p}{q}{\pri_Y}{\gtc{G}_Y \wrt (p,q)})}_{Y \in \widetilde{Y_X}})
                                            \Big)
                                        }_{q \in \tilde{q}_X}
                                    \end{array}
                                \right)
                            \right)
                        }_{X \in \widetilde{X_{\globCtx{C}}}},
                        \\
                        &Z:\left(
                            \begin{array}{@{}l@{}}
                                \deepUnfold(\ol{\gtc{G'} \locProjD^{\pri_{\globCtx{C}}} p}, {(X,t_X,\ol{\gtc{G}_X \locProjD^{\pri_X} p})}_{X \in \widetilde{X_{\globCtx{C}}}}),
                                \\[4pt]
                                {
                                    \Big(
                                        \deepUnfold(\relProp{p}{q}{\pri_{\globCtx{C}}}{\gtc{G'} \wrt (p,q)}, {(X,t_X,\relProp{p}{q}{\pri_X}{\gtc{G}_X \wrt (p,q)})}_{X \in \widetilde{X_{\globCtx{C}}}})
                                    \Big)
                                }_{q \in \tilde{q}'}
                            \end{array}
                        \right)
                        \\
                        &\vdash \lmed{p}{\tilde{q}'}{\gtc{G'}} \typInf \begin{array}[t]{@{}l@{}}
                            \ci{\mu}{p}:\unfold^{t_Z}\big(\mu Z . \deepUnfold(\ol{\gtc{G'} \locProjD^{\pri_{\globCtx{C}}} p}, {(X,t_X,\ol{\gtc{G}_X \locProjD^{\pri_X} p})}_{X \in \widetilde{X_{\globCtx{C}}}})\big),
                            \\[6pt]
                            {
                                \Big(
                                    \crt{p}{q}:\unfold^{t_Z}\big( \begin{array}[t]{@{}l@{}}
                                        \mu Z . \deepUnfold(\relProp{p}{q}{\pri_{\globCtx{C}}}{\gtc{G'} \wrt (p,q)}, \\
                                        {(X,t_X,\relProp{p}{q}{\pri_X}{\gtc{G}_X \wrt (p,q)})}_{X \in \widetilde{X_{\globCtx{C}}}})\big)
                                    \Big)
                                \end{array}
                            }_{q \in \tilde{q}'}
                        \end{array}
                    \end{align*}

                    By assumption, we have
                    \[
                        t_Z = \maxpr\left(
                            \begin{array}{@{}l@{}}
                                \deepUnfold(\ol{\gtc{G'} \locProjD^{\pri_{\globCtx{C}}} p}, {(X,t_X,\ol{\gtc{G}_X \locProjD^{\pri_X} p})}_{X \in \widetilde{X_{\globCtx{C}}}})
                                \\[4pt]
                                {
                                    \Big(
                                        \deepUnfold(\relProp{p}{q}{\pri_{\globCtx{C}}}{\gtc{G'} \wrt (p,q)}, {(X,t_X,\relProp{p}{q}{\pri_X}{\gtc{G}_X \wrt (p,q)})}_{X \in \widetilde{X_{\globCtx{C}}}})
                                    \Big)
                                }_{q \in \tilde{q}'}
                            \end{array}
                        \right) + 1,
                    \]
                    so $t_Z$ is clearly greater than the maximum priority appearing in the types before unfolding.
                    Hence, we can apply Rule~\ruleLabel{typ-rec} to eliminate $Z$ from the recursive context, and to fold the types, giving the typing of $\lmed{p}{\tilde{q}}{\gtc{G_s}} = \pRec Z(\ci{\mu}{p}, {(\crt{p}{q})}_{q \in \tilde{q}'}) ; \lmed{p}{\tilde{q}'}{\gtc{G'}}$:
                    \begin{align*}
                        &{
                            \left(
                                X:\left(
                                    \begin{array}{@{}l@{}}
                                        \deepUnfold(\ol{\gtc{G}_X \locProjD^{\pri_X} p}, {(Y,t_Y,\ol{\gtc{G}_Y \locProjD^{\pri_Y} p})}_{Y \in \widetilde{Y_X}}),
                                        \\[4pt]
                                        { \Big( \deepUnfold(\relProp{p}{q}{\pri_X}{\gtc{G}_X \wrt (p,q)}, {(Y,t_Y,\relProp{p}{q}{\pri_Y}{\gtc{G}_Y \wrt (p,q)})}_{Y \in \widetilde{Y_X}}) \Big) }_{q \in \tilde{q}_X}
                                    \end{array}
                                \right)
                            \right)
                        }_{X \in \widetilde{X_{\globCtx{C}}}}
                        \\
                        &\vdash \lmed{p}{\tilde{q}}{G_s} \typInf \begin{array}[t]{@{}l@{}}
                            \ci{\mu}{p}:\mu Z . \deepUnfold(\ol{\gtc{G'} \locProjD^{\pri_{\globCtx{C}}} p}, {(X,t_X,\ol{\gtc{G}_X \locProjD^{\pri_X} p})}_{X \in \widetilde{X_C}}),
                            \\[6pt]
                            { \Big( \crt{p}{q}:\mu Z . \deepUnfold(\relProp{p}{q}{\pri_{\globCtx{C}}}{\gtc{G'} \wrt (p,q)}, {(X,t_X,\relProp{p}{q}{\pri_X}{\gtc{G}_X \wrt (p,q)})}_{X \in \widetilde{X_{\globCtx{C}}}}) \Big) }_{q \in \tilde{q}'}
                        \end{array}
                    \end{align*}

                    In this typing, the type for $\ci{\mu}{p}$ concurs with~\eqref{eq:mpstAPCP:rRecDefMu}, and, for every $q \in \tilde{q}'$, the type for $\crt{p}{q}$ concurs with~\eqref{eq:mpstAPCP:rRecDefqprime}.
                    For every $q \in \tilde{q} \setminus \tilde{q}'$, we can add the type for $\crt{p}{q}$ in~\eqref{eq:mpstAPCP:rRecDefOthers} by applying Rule~\ruleLabel{typ-end}.
                    This proves the thesis.
            \end{itemize}

        \item
            \emph{Recursive call}: $\gtc{G_s} = Z$ (\cref{li:mpstAPCP:rtrRecCall}).

            Clearly, because $\gtc{G}$ is closed (i.e.\ $\frv(\gtc{G}) = \emptyset$), $Z \in \widetilde{X_{\globCtx{C}}}$.
            More precisely, $\widetilde{X_{\globCtx{C}}} = (\tilde{X}_1, Z, \tilde{X}_2)$.

            Note that the recursive definitions on the variables in $\tilde{X}_1$ appear in $\gtc{G}$ after the recursive definitions on the variables in $(Z, \tilde{X}_2)$.
            Because the unfoldings of $(Z, \tilde{X}_2)$ occur before the unfoldings of $\tilde{X}_1$, the recursive definitions on the variables in $\tilde{X}_1$ are renamed in order to avoid capturing these variables when performing the unfoldings of $(Z, \tilde{X}_2)$.
            So, after the unfoldings of $(Z, \tilde{X}_2)$, there are no recursive calls on the variables in $\tilde{X}_1$ anymore, so the unfoldings on $\tilde{X}_1$ do not have any effect on the types.

            Also, note that $\tilde{X}_2 = \widetilde{Y_Z}$ (cf.\ \Cref{d:mpstAPCP:binders}).

            Let us take stock of the types we expect for our router's channels.
            \begin{align}
                & \text{For $\ci{\mu}{p}$ we expect}
                &
                &\deepUnfold(\ol{\gtc{G_s} \locProjD^{\pri_{\globCtx{C}}} p},\ldots)
                \nonumber \\
                &
                &
                &{} = \deepUnfold(\ol{Z \locProjD^{\pri_{\globCtx{C}}} p},\ldots)
                \nonumber \\
                &
                &
                &{} = \deepUnfold(\ol{Z},\ldots)
                \nonumber \\
                &
                &
                &{} = \deepUnfold(Z, {(X,t_X,\ol{\gtc{G}_X \locProjD^{\pri_X} p})}_{X \in (\tilde{X}_1, Z, \widetilde{Y_Z})})
                \nonumber \\
                &
                &
                &{} = \deepUnfold(Z, {(X,t_X,\ol{\gtc{G}_X \locProjD^{\pri_X} p})}_{X \in (Z, \widetilde{Y_Z})})
                \nonumber \\
                &
                &
                &{} = \mu Z . (\lift{t_Z} \deepUnfold(\ol{\gtc{G}_Z \locProjD^{\pri_Z} p}, {(X,t_X,\ol{\gtc{G}_X \locProjD^{\pri_X} p})}_{X \in \widetilde{Y_Z}}))
                \label{eq:mpstAPCP:rCallMu}
            \end{align}
            \begin{align}
                & \text{For each $q \in \tilde{q}$,}
                \nonumber \\
                & \text{for $\crt{p}{q}$ we expect}
                &
                &\deepUnfold(\relProp{p}{q}{\pri_{\globCtx{C}}}{\gtc{G_s} \wrt (p,q)},\ldots)
                \nonumber \\
                &
                &
                &{} = \deepUnfold(\relProp{p}{q}{\pri_{\globCtx{C}}}{Z},\ldots)
                \nonumber \\
                &
                &
                &{} = \deepUnfold(Z, {(X,t_X,\relProp{p}{q}{\pri_X}{\gtc{G}_X \wrt (p,q)})}_{X \in (\tilde{X}_1, Z, \widetilde{Y_Z})})
                \nonumber \\
                &
                &
                &{} = \deepUnfold(Z, {(X,t_X,\relProp{p}{q}{\pri_X}{\gtc{G}_X \wrt (p,q)})}_{X \in (Z, \widetilde{Y_Z})})
                \nonumber \\
                &
                &
                &{} = \mu Z . (\lift{t_Z} \deepUnfold(\relProp{p}{q}{\pri_Z}{\gtc{G}_Z \wrt (p,q)}, {(X,t_X,\relProp{p}{q}{\pri_X}{\gtc{G}_X \wrt (p,q)})}_{\hspace{-2pt}X \in \widetilde{Y_Z}}))
                \label{eq:mpstAPCP:rCallTildeq}
            \end{align}

            Also, we need an assignment in the recursive context for every $X \in \widetilde{X_{\globCtx{C}}}$.
            By \Cref{l:mpstAPCP:binderParts}, $\tilde{q} = \tilde{q}_Z$.
            Hence, for $Z$, the assignment should be as follows:
            \begin{align}
                Z:\left(
                    \begin{array}{@{}l@{}}
                        \deepUnfold(\ol{\gtc{G}_Z \locProjD^{\pri_Z} p}, {(X,t_X,\ol{\gtc{G}_X \locProjD^{\pri_X} p})}_{X \in \widetilde{Y_Z}}),
                        \\
                        { \Big( \deepUnfold(\relProp{p}{q}{\pri_Z}{\gtc{G}_Z \wrt (p,q)}, {(X,t_X,\relProp{p}{q}{\pri_X}{\gtc{G}_X \wrt (p,q)})}_{X \in \widetilde{Y_Z}}) \Big) }_{q \in \tilde{q}}
                    \end{array}
                \right)
                \label{eq:mpstAPCP:rCallRec}
            \end{align}

            We apply Rule~\ruleLabel{typ-var} to obtain the typing of $\lmed{\tilde{q}}{p}{\gtc{G_s}}$, where we make us the rule's allowance for an arbitrary recursive context up to the assignment to $Z$.
            Rule~\ruleLabel{typ-var} is applicable, because the types are recursive definitions on $Z$, concurring with the types assigned to $Z$, and lifted by a common lifter $t_Z$.
            \[
                \begin{bussproof}
                    \bussAx{
                        \begin{array}{@{}l@{}}
                            {
                                \left(
                                    X:\left(
                                        \begin{array}{@{}l@{}}
                                            \deepUnfold(\ol{\gtc{G}_X \locProjD^{\pri_X} p}, {(Y,t_Y,\ol{\gtc{G}_Y \locProjD^{\pri_Y} p})}_{Y \in \widetilde{Y_X}}),
                                            \\[4pt]
                                            { \Big( \deepUnfold(\relProp{p}{q}{\pri_X}{\gtc{G}_X \wrt (p,q)}, {(Y,t_Y,\relProp{p}{q}{\pri_Y}{\gtc{G}_Y \wrt (p,q)})}_{Y \in \widetilde{Y_X}}) \Big) }_{q \in \tilde{q}_X}
                                        \end{array}
                                    \right)
                                \right)
                            }_{X \in \widetilde{X_{\globCtx{C}}} \setminus (Z)},
                            \\
                            Z:\left(
                                \begin{array}{@{}l@{}}
                                    \deepUnfold(\ol{\gtc{G}_Z \locProjD^{\pri_Z} p}, {(X,t_X,\ol{\gtc{G}_X \locProjD^{\pri_X} p})}_{X \in \widetilde{Y_Z}}),
                                    \\
                                    { \Big( \deepUnfold(\relProp{p}{q}{\pri_Z}{\gtc{G}_Z \wrt (p,q)}, {(X,t_X,\relProp{p}{q}{\pri_X}{\gtc{G}_X \wrt (p,q)})}_{X \in \widetilde{Y_Z}}) \Big) }_{q \in \tilde{q}}
                                \end{array}
                            \right)
                            \\
                            \vdash \lmed{\tilde{q}}{p}{\gtc{G_s}} = \pCall X<\ci{\mu}{p}, {(\crt{p}{q})}_{q \in \tilde{q}}>
                            \\
                            \typInf \begin{array}[t]{@{}l@{}}
                                \ci{\mu}{p}:\mu Z . (\lift{t_Z} \deepUnfold(\ol{\gtc{G}_Z \locProjD^{\pri_Z} p}, {(X,t_X,\ol{\gtc{G}_X \locProjD^{\pri_X} p})}_{X \in \widetilde{Y_Z}})),
                                \\[6pt]
                                { \left( \crt{p}{q}:\mu Z . (\lift{t_Z} \deepUnfold(\relProp{p}{q}{\pri_Z}{\gtc{G}_Z \wrt (p,q)}, {(X,t_X,\relProp{p}{q}{\pri_X}{\gtc{G}_X \wrt (p,q)})}_{X \in \widetilde{Y_Z}})) \right) }_{q \in \tilde{q}}
                            \end{array}
                        \end{array}
                    }
                \end{bussproof}
            \]
            In this typing, the type of $\ci{\mu}{p}$ concurs with the expected type in~\eqref{eq:mpstAPCP:rCallMu}, the types of $\crt{p}{q}$ for each $q \in \tilde{q}$ concur with the expected types in~\eqref{eq:mpstAPCP:rCallTildeq}, and the assignment to $Z$ in the recursive context concurs with~\eqref{eq:mpstAPCP:rCallRec}.
            This proves the thesis.
            \qedhere
    \end{itemize}
\end{proof}

\tMpstAPCPRouterTypes*

\ifmpstAPCP
\begin{proof}
\else
\begin{proof}[Proof of \Cref{t:mpstAPCP:routerTypes} on \Cpageref{t:mpstAPCP:routerTypes}]
\fi
    \label{proof:mpstAPCP:routerTypes}
    We have been given a closed, relative well-formed global type $\gtc{G}$, and a participant \mbox{$p \in \prt(\gtc{G})$}.
    Let $\globCtx{C} \deq \globHole$ and $\gtc{G_s} \deq \gtc{G}$.
    Clearly, $\gtc{G_s} \gSub{\globCtx{C}} \gtc{G}$.
    By \Cref{d:mpstAPCP:activeParts}, ${\act(\globCtx{C},\gtc{G}) = {\prt(\gtc{G})}^2}$.
    For $p$ to be a participant of $\gtc{G}$, there must be an exchange involving~$p$ and some other participant $q$, i.e.\ there exists a $q \in \prt(\gtc{G})$ such that $(p,q) \in \act(\globCtx{C},\gtc{G})$.
    Moreover, $\tilde{q}$ as defined in \Cref{t:mpstAPCP:routerTypesGen} is
    \[
        \{ q \in \prt(\gtc{G}) \mid (p,q) \in \act(\globCtx{C},\gtc{G}) \} = q \in \prt(\gtc{G}) \setminus \{p\}.
    \]
    Hence, \Cref{t:mpstAPCP:routerTypesGen} allows us to find a typing for $\lmed{p}{\prt(\gtc{G}) \setminus \{p\}}{\gtc{G}}$.

    Let us consider the precise values of the ingredients of \Cref{t:mpstAPCP:routerTypesGen} in our application:
    \begin{enumerate}
        \item
            $\pri_{\globCtx{C}} = \ctxPri(\globCtx{C}) = 0$,

        \item
            $\widetilde{X_{\globCtx{C}}} = \ctxBind(\globCtx{C}) = ()$,

        \item
            $\begin{array}[t]{@{}r@{}ll@{}}
                D_p
                &{}= \deepUnfold(\ol{\gtc{G_s} \locProjD^{\pri_{\globCtx{C}}} p}, {(X,t_X,\ol{\gtc{G}_x \locProjD^{\pri_X} p})}_{X \in \widetilde{X_{\globCtx{C}}}})
                \\
                &{}= \ol{\gtc{G} \locProjD^0 p}
                &\text{(cf.\ \Cref{d:mpstAPCP:deepUnfold}),}
            \end{array}$

        \item
            $\begin{array}[t]{@{}r@{}ll@{}}
                E_q
                &{}= \deepUnfold(\relProp{p}{q}{\pri_{\globCtx{C}}}{\gtc{G_s} \wrt (p,q)}, {(X,t_X,\relProp{p}{q}{\pri_X}{\gtc{G}_X \wrt (p,q)})}_{X \in \widetilde{X_{\globCtx{C}}}})
                \\
                &{}= \relProp{p}{q}{0}{\gtc{G} \wrt (p,q)}
                &\text{(cf.\ \Cref{d:mpstAPCP:deepUnfold}).}
            \end{array}$
    \end{enumerate}

    Finally, the result of \Cref{t:mpstAPCP:routerTypesGen} is as follows:
    \[
        { \Big( X:\big( A_X, {(B_{X,q})}_{q \in \tilde{q}_X} \big) \Big) }_{X \in \widetilde{X_{\globCtx{C}}}}
        \vdash \lmed{p}{\tilde{q}}{\gtc{G_s}} \typInf
        \ci{\mu}{p}:D_p,
        {(\crt{p}{q}:E_q)}_{q \in \tilde{q}}.
    \]
    Applying (1)--(4) above, we get the following:
    \[
        \emptyset
        \vdash \lmed{p}{\prt(\gtc{G}) \setminus \{p\}}{\gtc{G}} \typInf
        \ci{\mu}{p}:\ol{\gtc{G} \locProjD^0 p},
        { \big( \crt{p}{q}:\relProp{p}{q}{0}{\gtc{G} \wrt (p,q)} \big) }_{q \in \prt(\gtc{G}) \setminus \{p\}}.
    \]
    This coincides exactly with the result of \Cref{t:mpstAPCP:routerTypes}.
\end{proof}
\fi

\subsubsection{Transference of Results (Operational Correspondence)}
\label{s:mpstAPCP:analysis:transfer}

Given a global type $\gtc{G}$, we now formalize the transference of correctness properties such as deadlock-freedom from $\sys(\gtc{G})$ (cf.\ \Cref{d:mpstAPCP:networks}) to $\gtc{G}$.
Here, we define an \emph{operational correspondence} between networks and global types, in both directions.
That is, we show that a network performs interactions between implementations and routers and between pairs of routers if and only if that communication step is stipulated in the corresponding global type (\Cref{t:mpstAPCP:completeness,t:mpstAPCP:soundness}).

\paragraph{Alarm-freedom.}

Before formalizing the operational correspondence, we show that networks of routed implementations never reduce to alarm processes.
To be precise, because alarm processes only can occur in routers (not in implementations), we show that
none of the routers of a network reduce to an alarm process, formalized using evaluation contexts\ifmpstAPCP\else (\Cref{d:mpstAPCP:evalCtx})\fi:

\ifmpstAPCP
\ExecuteMetaData[\fileLMCS]{mpstAPCP:APCP:evalCtx}
\fi

\begin{restatable}{theorem}{tMpstAPCPNoAlarm}
    \label{t:mpstAPCP:noAlarm}
    Given a relative well-formed global type $\gtc{G}$ and a network of routed implementations $\mcl{N} \in \sys(\gtc{G})$, then
    \[
        \mcl{N} \mathrel{\nredd^\ast} \evalCtx{E}[\pAlarm(\tilde{x})],
    \]
    for any evaluation context $\evalCtx{E}$ and set of endpoints $\tilde{x}$.
\end{restatable}

\ifmpstAPCP
\begin{sketch}
    By contradiction: alarm processes only occur when a router receives different labels for the same dependency, but routers are hard-coded to forward the correct label.
\end{sketch}
\else
\begin{proof}
    By \Cref{d:mpstAPCP:networks}, $\mcl{N}$ consists only of routers (\Cref{d:mpstAPCP:router}) and well-typed processes not containing the alarm process (cf.\ the assumption below \Cref{d:mpstAPCP:alarm}).

    Suppose, for contradiction, that there are $\evalCtx{E}$ and $\tilde{x}$ such that $\mcl{N} \redd* \evalCtx{E}[\pAlarm(\tilde{x})]$.
    Since only routers can contain the alarm process, there is a router $\rtr{p}$ in $\mcl{N}$ for participant $p \in \prt(\gtc{G})$ that reduces to the alarm process.
    Since it is the only possibility for a router synthesized by \Cref{alg:mpstAPCP:router} to contain the alarm process, it must contain the process in \eqref{eq:mpstAPCP:patched}.
    This process is synthesized on \cref{li:mpstAPCP:rtrBoth} of \Cref{alg:mpstAPCP:router}, so there is an exchange in $\gtc{G}$ with sender $s \in \prt(\gtc{G}) \setminus \{p\}$ and recipient $r \in \prt(\gtc{G}) \setminus \{p\}$ that is a dependency for the interactions of $p$ with both $s$ and $r$.

    For this exchange, the router $\rtr{s}$ for $s$ contains the process returned on \cref{li:mpstAPCP:rtrSend} of \Cref{alg:mpstAPCP:router}, and the router $\rtr{r}$ for $r$ contains the process returned on \cref{li:mpstAPCP:rtrRecv}.
    Suppose $s$ has a choice between the labels in $I$, and the implementation of $s$ chooses $i \in I$.
    Then, $\rtr{s}$ sends $i$ to $\rtr{r}$ and $\rtr{p}$.

    Now, for $\rtr{p}$ to reduce to the alarm process, it has to receive from $\rtr{r}$ a label $i' \in I \setminus \{i\}$.
    However, this contradicts \cref{li:mpstAPCP:rtrRecv} of \Cref{alg:mpstAPCP:router}, which clearly defines $\rtr{r}$ to send $i$ to $\rtr{p}$.
    Hence, $\mcl{N} \mathrel{\nredd^\ast} \evalCtx{E}[\pAlarm(\tilde{x})]$.
\end{proof}
\fi

\paragraph{Deadlock-freedom.}

It follows from this and the typability of routers (\Cref{t:mpstAPCP:routerTypes}) that networks of routed implementations are deadlock-free:

\begin{theorem}
    \label{t:mpstAPCP:globalDf}
    For relative well-formed global type $\gtc{G}$, every $\mcl{N} \in \sys(\gtc{G})$ is deadlock-free.
\end{theorem}

\begin{proof}
    By the typability of routers (\Cref{t:mpstAPCP:routerTypes}) and the duality of the types of router channels (\Cref{t:mpstAPCP:relPropDual}), $\emptyset \vdash \mcl{N} \typInf \emptyset$.
    Hence, by \ifappendix \Cref{t:mpstAPCP:APCP:df}\else deadlock-freedom for \APCP (\Cref{t:APCP:df} on \Cpageref{t:APCP:df})\fi, $\mcl{N}$ is deadlock-free, and by \Cref{t:mpstAPCP:noAlarm}, $\mcl{N}$ never reduces to the alarm process.
\end{proof}

\paragraph{Operation correspondence.}

To formalize our operational correspondence result, we apply the labeled reductions for processes $P \lredd{\alpha} Q$ (cf.\ \ifappendix \Cref{d:mpstAPCP:APCP:procLred}\else \Cref{d:APCP:procLred} on \Cpageref{d:APCP:procLred}\fi) and define a labeled transition system (\LTS) for global types $\gtc{G} \lredd{\alpha} \gtc{G'}$.

\begin{definition}[\LTS for Global Types]
    \label{d:mpstAPCP:globalLts}
    We define the relation $\gtc{G} \lredd{\alpha} \gtc{G'}$, with labels $\beta$ of the form $p {!} q (\msg \ell<S>)$ (sender, recipient, label, and message type), by the following rules:
    \begin{mathpar}
        \begin{bussproof}[glob-exch]
            \bussAssume{
                j \in I
            }
            \bussUn{
                p {!} q \gtBraces{ \msg i<S_i> \gtc. \gtc{G_i} }_{i \in I} \lredd{p {!} q (\msg j<S_j>)} \gtc{G_j}
            }
        \end{bussproof}
        \and
        \begin{bussproof}[glob-skip]
            \bussAssume{
                \gtc{G} \lredd{\alpha} \gtc{G'}
            }
            \bussUn{
                \tSkip \gtc. \gtc{G} \lredd{\alpha} \gtc{G'}
            }
        \end{bussproof}
        \and
        \begin{bussproof}[glob-rec]
            \bussAssume{
                \gtc{G} \{ (\mu X \gtc. \gtc{G}) / X \} \lredd{\alpha} \gtc{G'}
            }
            \bussUn{
                \mu X \gtc. \gtc{G} \lredd{\alpha} \gtc{G'}
            }
        \end{bussproof}
    \end{mathpar}
\end{definition}

Intuitively, operational correspondence states:
\begin{enumerate}
    \item
        every transition of a global type is mimicked by a precise sequence of labeled reductions originating from a completable network (\emph{completeness}; \Cref{t:mpstAPCP:completeness}), and
    \item
        for every labeled reduction originated in a completable network there is a corresponding global type transition (\emph{soundness}; \Cref{t:mpstAPCP:soundness}).
\end{enumerate}

\noindent
We write $\rho_1 \rho_2$ for the composition of relations $\rho_1$ and $\rho_2$.
Recall that the notation $\redd^\star$ stands for finite sequences of reductions\ifmpstAPCP\else, as defined in \Cref{n:mpstAPCP:redd}\fi.

\paragraph{Completeness.}

\begin{restatable}[Operational Correspondence: Completeness]{theorem}{tMpstAPCPCompleteness}
    \label{t:mpstAPCP:completeness}
    Suppose given a relative well-formed global type $\gtc{G}$.
    Also, suppose given $p,q \in \prt(\gtc{G})$ and a set of labels $J$ such that~\mbox{$j \in J$} if and only if $\gtc{G} \lredd{p {!} q (\msg j<S_j>} \gtc{G_j}$ for some $S_j$.
    Then,
    \begin{enumerate}
        \item
            for any completable $\mcl{N} \in \sys(\gtc{G})$, there exists a $j' \in J$ such that $\mcl{N}^\complet \redd^\star \lredd{\rLbl \ci{p}{\mu}>\ci{\mu}{p}:j'} \mcl{N}_0$;
        \item
            for any $j' \in J$, there exists a completable $\mcl{N} \in \sys(\gtc{G})$ such that $\mcl{N}^\complet \redd^\star \lredd{\rLbl \ci{p}{\mu}>\ci{\mu}{p}:j'} \mcl{N}_0$;
        \item
            for any completable $\mcl{N} \in \sys(\gtc{G})$ and any $j' \in J$, if $\mcl{N}^\complet \redd^\star \lredd{\rLbl \ci{p}{\mu}>\ci{\mu}{p}:j'} \mcl{N}_0$, then there exists a comple\-table $\mcl{N}_{j'} \in \sys(\gtc{G_{j'}})$ such that,
            \[
                \mcl{N}_0 \lredd{\rLbl \crt{p}{q}>\crt{q}{p}:j'} \redd^\star \lredd{\rLbl \ci{\mu}{q}>\ci{q}{\mu}:j'} \redd^\star \lredd{\rLbl \ci{p}{\mu}>\ci{\mu}{p}:v} \lredd{\rLbl \crt{p}{q}>\crt{q}{p}:w} \lredd{\pFwd [v<>w]} \redd^\star \lredd{\rLbl \ci{\mu}{q}>\ci{q}{\mu}:w} \lredd{\pFwd [v<>w]} \mcl{N}_{j'}^\complet.
            \]
    \end{enumerate}
\end{restatable}

\ifmpstAPCP
\begin{sketch}
    Each Sub-item is proved separately:
    \begin{enumerate}
        \item
            By definition of completable networks (\Cref{d:mpstAPCP:complNet}), routers (\Cref{d:mpstAPCP:router}), and their types (\Cref{t:mpstAPCP:routerTypes}), the thesis follows from \ifappendix reactivity (\Cref{t:mpstAPCP:APCP:react})\else \APCP's reactivity result (\Cref{t:APCP:react} on \Cpageref{t:APCP:react})\fi.
        \item
            We can construct an implementation process for $p$ that picks $j' \in J$\ifappendix (cf.\ \Cref{p:mpstAPCP:charProc})\fi.
            The thesis then follows from Sub-item~(1).
        \item
            The thesis follows by the shape of $\gtc{G}$ and the routers in $\mcl{N}$ (cf.\ \Cref{d:mpstAPCP:router}).
            \qedhere
    \end{enumerate}
\end{sketch}
\else
\begin{proof}
    By the labelled transitions of global types (\Cref{d:mpstAPCP:globalLts}) and relative well-formedness, $\gtc{G}$ is a sequence of $\tSkip$s followed by an exchange from $p$ to $q$ over the labels in $J$.
    Since the $\tSkip$s do not influence the behavior of routers, let us assume simply that
    \[
        \gtc{G} = p {!} q \gtBraces{ \msg j<S_j> \gtc. \gtc{G_j} }_{j \in J}.
    \]
    We prove each Sub-item separately.

    \begin{enumerate}
        \item
            Take any completable $\mcl{N} \in \sys(\gtc{G})$.
            By definition (\Cref{d:mpstAPCP:complNet}), $\emptyset \vdash \mcl{N}^\complet \typInf \emptyset$.
            By the construction of networks of routed implementations (\Cref{d:mpstAPCP:networks}), $\ci{p}{\mu} \in \bn(\mcl{N}^\complet)$, and $\ci{p}{\mu}$ is connected to $\ci{\mu}{p}$.

            Also by construction, the type of $\ci{p}{\mu}$ in the typing derivation of $\mcl{N}^\complet$ is
            \[
                \gtc{G} \locProjD^0 p = \oplus^0 \{ j : \msgProp{S_j} \tensor^1 (\gtc{G_j} \locProjD^4 p) \}_{j \in J}.
            \]
            By the well-typedness of $\mcl{N}^\complet$, we can infer the kind of action that is defined on $\ci{p}{\mu}$: a selection, or a forwarder.
            By induction on the number of connected forwarders (which is finite by the finiteness of process terms), eventually a forwarder has to be connected to a selection.
            So, after reducing the forwarders, we have a selection on $\ci{p}{\mu}$, of some~\mbox{$j' \in J$}.

            Hence, by \ifappendix reactivity (\Cref{t:mpstAPCP:APCP:react})\else \APCP's reactivity (\Cref{t:APCP:react} on \Cpageref{t:APCP:react})\fi, after a finite number of steps, we can observe a communication of the label $j'$ from $\ci{p}{\mu}$ to $\ci{\mu}{p}$.
            This proves the thesis:
            \[
                \mcl{N}^\complet \redd^\star \lredd{\rLbl \ci{p}{\mu}>\ci{\mu}{p}:j'} \mcl{N}_0.
            \]

        \item
            Following the proof of the existence of completable networks (\Cref{p:mpstAPCP:complNetsExist}), we can generate an implementation process for all of $\gtc{G}$'s participants from local projections (cf.\ \Cref{p:mpstAPCP:charProc}).
            Take any $j' \in J$.
            For the implementation process of $p$, we specifically generate an implementation process that sends the label $j'$.
            These implementation processes allow us to construct $\mcl{N}$, which by construction is in $\sys(\gtc{G})$ and is completable.
            Following the reasoning as in Sub-item~(1), ${\mcl{N}^\complet \redd^\star \lredd{\rLbl \ci{p}{\mu}>\ci{\mu}{p}:j'} \mcl{N}_0}$.

        \item
            By definition (\Cref{d:mpstAPCP:complNet}), $\emptyset \vdash \mcl{N}^\complet \typInf \emptyset$.
            Hence, by \ifappendix reactivity (\Cref{t:mpstAPCP:APCP:react})\else \APCP's reactivity (\Cref{t:APCP:react} on \Cpageref{t:APCP:react})\fi, for any of the pending names of $\mcl{N}^\complet$, we can observe a communication after a finite number of steps.
            By construction (\Cref{d:mpstAPCP:networks}), the endpoints that we are required to observe by thesis are bound in $\mcl{N}^\complet$.
            From the shape of $\gtc{G}$, the definition of routed implementations (\Cref{d:mpstAPCP:routedImplementations}), and the typability of routers (\Cref{t:mpstAPCP:routerTypes}), we know the types of all the required endpoints in $\mcl{N}^\complet$.
            We can deduce the required labeled reductions following the reasoning as in Sub-item~(1).
            Let us summarize the origin of each of the network's steps:
            \begin{enumerate}[label=\arabic*.]
                \item
                    $\mcl{N}^\complet \redd^\star \lredd{\rLbl \ci{p}{\mu}>\ci{\mu}{p}:j'} \mcl{N}_0$:
                    The implementation of $p$ selects label $j'$ with $p$'s router.

                \item
                    $\mcl{N}_0 \lredd{\rLbl \crt{p}{q}>\crt{q}{p}:j'} \mcl{N}_1$:
                    The router of $p$ forwards $j'$ to $q$'s router.

                \item
                    $\mcl{N}_1 \redd^\star \mcl{N}_2$:
                    The router of $p$ forwards $j'$ to the routers of the participant that depend on the output by $p$, and these routers forward $j'$ to their respective implementations.

                \item
                    $\mcl{N}_2 \lredd{\rLbl \ci{\mu}{q}>\ci{q}{\mu}:j'} \redd^\star \mcl{N}_3$:
                    The router of $q$ forwards $j'$ to $q$'s implementation, and to the routers of the participants that depend on the input by $q$, and these routers forward $j'$ to their respective implementation (if they have not done so already for the output dependency on $p$).

                \item
                    $\mcl{N}_3 \lredd{\rLbl \ci{p}{\mu}>\ci{\mu}{p}:v} \lredd{\rLbl \crt{p}{q}>\crt{q}{p}:w} \lredd{\pFwd [v<>w]} \mcl{N}_{4,v}$:
                    The implementation of $p$ sends an endpoint $v$ to $p$'s router, which sends a fresh endpoint $w$ to $q$'s router, and $v$ is forwarded to $w$.

                \item
                    $\mcl{N}_{4,v} \redd^\star \lredd{\rLbl \ci{\mu}{q}>\ci{q}{\mu}:w} \lredd{\pFwd [v<>w]} \mcl{N}_{j'}^\complet$:
                    The router of $q$ sends a fresh endpoint $w$ to $q$'s implementations, and $v$ is forwarded to $w$.

                    In $\mcl{N}_{j'}^\complet$, all routers have transitioned to routers for $\gtc{G_{j'}}$.
                    Moreover, by \ifappendix type preservation (\Cref{t:mpstAPCP:APCP:tp})\else \APCP's type preservation (\Cref{t:APCP:sr} on \Cpageref{t:APCP:sr})\fi, $\emptyset \vdash \mcl{N}_{j'}^\complet \typInf \emptyset$.
                    By isolating restrictions on endpoints that belong only to implementation processes, we can find $\mcl{N}_{j'} \in \sys(\gtc{G_{j'}})$ such that $\mcl{N}_{j'}^\complet$ is its completion.
                    This proves the thesis.
            \end{enumerate}
    \end{enumerate}

    Note that $\gtc{G}$ can also contain recursive definitions before the initial exchange; this case can be dealt with by unfolding.
\end{proof}
\fi

\paragraph{Soundness.}

Our soundness result, given below as~\Cref{t:mpstAPCP:soundness}, will  capture the notion that after any sequence of reductions from the network of a global type $\gtc{G}$, a network of another global type $\gtc{G'}$ can be reached.
Crucially, $\gtc{G'}$ can be reached from $\gtc{G}$ through a series of transitions.
Networks are inherently concurrent, whereas global types are built out of sequential compositions; as a result, the network could have enabled (asynchronous) actions that correspond to exchanges that are not immediately enabled in the global type.

\begin{example}
    \label{x:mpstAPCP:indep}
    Consider the global types
    \[
        \gtc{G} \deq a {!} b \gtBraces[\big]{ \msg 1<S_1> \gtc. c {!} d \gtBraces{ \msg 1<S'> \gtc. \tEnd}, \msg 2<S_2> \gtc. c {!} d \gtBraces{ \msg 1<S'> \gtc. \tEnd } }
    \]
    and
    \[
        \gtc{G'} \deq a {!} b \gtBraces[\big]{ \msg 1<S> \gtc. b {!} c \gtBraces{ \msg 1<S'> \gtc. \tEnd } }.
    \]
    Clearly, the initial exchange in $\gtc{G}$ between $a$ and $b$ is not a dependency for the following exchange between $c$ and $d$.
    The routers of $c$ and $d$ synthesized from $\gtc{G}$ thus start with their exchange, without awaiting the initial exchange between $a$ and $b$ to complete.
    Hence, in a network of $\gtc{G}$, both exchanges in $\gtc{G}$ may be enabled simultaneously.
\end{example}

We further refer to exchanges that may be simultaneously enabled in networks as \emph{independent (global) exchanges}.
In \Cref{x:mpstAPCP:indep} above, all exchanges appearing in $\gtc{G}$ are independent, while the two exchanges in $\gtc{G'}$ are not.

\ifmpstAPCP\else
In the proof of soundness, we may encounter in a network reductions related to independent exchanges, so we have to be able to identify the independent exchanges in the global type to which the network belongs.
\Cref{l:mpstAPCP:delayedPrefix} states that independent exchanges related to observed reductions in a network of a global type $\gtc{G}$ can be reached from $\gtc{G}$ after any sequence of transitions in a finite number of steps.
The proof of this lemma relies on \Cref{l:mpstAPCP:routerNodepEqual}, which ensures that if a participant does not depend on a certain exchange, then the routers synthesized at each of the branches of the exchange are equal.

\begin{lemma}
    \label{l:mpstAPCP:routerNodepEqual}
    Suppose given a relative well-formed global type $\gtc{G} = s {!} r \gtBraces{ \msg i<S_i> \gtc. \gtc{G_i} }_{i \in I}$, and take any ${p \in \prt(\gtc{G}) \setminus \{s,r\}}$ and $\tilde{q} \subseteq \prt(\gtc{G}) \setminus \{p\}$.
    If neither $\hdep(p,s,\gtc{G})$ nor $\hdep(p,r,\gtc{G})$ holds, then ${\lmed{p}{\tilde{q}}{\gtc{G_i}} = \lmed{p}{\tilde{q}}{\gtc{G_j}}}$ for every $i,j \in I$.
\end{lemma}

\begin{proof}
    The analysis proceeds by cases on the structure of $\gtc{G}$.
    As a representative case we consider ${\gtc{G} = s {!} r \gtBraces{ \msg 1<S_1> \gtc. \gtc{G_1}, \msg 2<S_2> \gtc. \gtc{G_2} }}$.
    Towards a contradiction, we assume \mbox{$\lmed{p}{\tilde{q}}{\gtc{G_1}} \neq \lmed{p}{\tilde{q}}{\gtc{G_2}}$}.
    There are many cases where \Cref{alg:mpstAPCP:router} generates differents routers for $p$ at~$\gtc{G_1}$ and at~$\gtc{G_2}$.
    We discuss the interesting case where $\lmed{p}{\tilde{q}}{\gtc{G_1}} = \pBra* \ci{\mu}{p} > \ldots$~(\cref{li:mpstAPCP:rtrSend}) and \mbox{$\lmed{p}{\tilde{q}}{\gtc{G_2}} = \pBra* \crt{p}{q_2} > \ldots$}~(\cref{li:mpstAPCP:rtrRecv}).
    Then $\gtc{G_1} = p {!} q_1 \gtBraces{ \ldots }$ and $\gtc{G_2} = q_2 {!} p \gtBraces{ \ldots }$.
    We have \mbox{$\gtc{G_1} \wrt (p,q_1) = p {!} q_1 \rtBraces{ \ldots }$} and \mbox{$\gtc{G_2} \wrt (p,q_1) = \tSkip \ldots$} or $\gtc{G_2} \wrt (p,q_1) = (p {?} q_2) {!} q_1 \rtBraces{ \ldots }$ (w.l.o.g., assume the former).
    Since $\gtc{G}$ is relative well-formed, the projection $\gtc{G} \wrt (p,q_1)$ must exist.
    Hence, since $p \notin \{s,r\}$ and \mbox{$\gtc{G_1} \wrt (p,q_1) \neq \gtc{G_2} \wrt (p,q_1)$}, it must be the case that \mbox{$q_1 \in \{s,r\}$}---w.l.o.g., assume $q_1 = s$.
    Then \mbox{$\gtc{G} \wrt (p,q_1) = (q_1 {!} r) {!} p \rtBraces{ 1 \rtc. p {!} q_1 \rtBraces{ \ldots }, 2 \rtc. \tSkip \ldots }$}, and thus \mbox{$\hdep(p,q_1,\gtc{G}) = \hdep(p,s,\gtc{G})$} is true.
    This contradicts the assumption that $\hdep(p,s,\gtc{G})$ is false.
\end{proof}

\begin{lemma}\label{l:mpstAPCP:delayedPrefix}
    Suppose given a relative well-formed global type $\gtc{G}$ and a completable \mbox{$\mcl{N} \in \sys(\gtc{G})$} such that $\mcl{N}^\complet \lredd{\rLbl \ci{c}{\mu}>\ci{\mu}{c}:\ell}$, for some $c \in \prt(\gtc{G})$.
    For every $\gtc{G'}$ and $\beta_1,\ldots,\beta_n$ ($n \geq 0$) such that
    $\gtc{G} \lredd{\beta_1} \ldots \lredd{\beta_n} \gtc{G'}$ where $c$ is not involved in any $\beta_k$ (with $\gtc{G} = \gtc{G'}$ if $n=0$),
    there exist $\gtc{G''}$, $d \in \prt(\gtc{G})$, and $\beta'_1,\ldots,\beta'_m$ ($m \geq 0$)
    such that ${\gtc{G'} \lredd{\beta'_1} \ldots \lredd{\beta'_m} \gtc{G''} = c {!} d \gtBraces{ \msg i<S_i> \gtc. \gtc{G_i} }_{i \in I}}$ where $c$ is not involved in any $\beta'_k$  (with $\gtc{G''} = \gtc{G'}$ if $m=0$).
\end{lemma}

\begin{proof}
    By induction on $n$ (\ih1).
    We first observe that the behavior on $\ci{\mu}{c}$ in $\mcl{N}^\complet$ can only arise from the router generated for $c$ at $\gtc{G}$, following \Cref{alg:mpstAPCP:router} (\cref{li:mpstAPCP:rtrSend}) after finitely many passes through lines~\labelcref{li:mpstAPCP:rtrElseElse} (no dependency) and \labelcref{li:mpstAPCP:rtrSkip} (skip); for simplicity, assume only \cref{li:mpstAPCP:rtrElseElse} applies.

    \begin{itemize}
        \item Case $n = 0$.
            Let $x \geq 0$ denote the number of passes through \cref{li:mpstAPCP:rtrElseElse} to generate the router for $c$ at~$\gtc{G}$.
            We apply induction on $x$ (\ih2):
            \begin{itemize}
                \item Case $x = 0$.
                    The router for $c$ at $\gtc{G}$ is generated through \cref{li:mpstAPCP:rtrSend}, so $\gtc{G} = c {!} d \gtBraces{ \msg i<S_i> \gtc. \gtc{G_i} }_{i \in I}$, proving the thesis.

                \item Case $x = x'+1$.
                    Then $\gtc{G} = a {!} b \gtBraces{ \msg i<S_i> \gtc. \gtc{G_i} }_{i \in I}$ and \cref{li:mpstAPCP:rtrElseElse} returns the router for $c$ at~$\gtc{G_j}$ for any $j \in I$.
                    We have $\gtc{G} \lredd{a {!} b (\msg j<S_j>)} \gtc{G_j}$.
                    Given the same implementation process for $c$ as in $\mcl{N}$, we can construct a completable $\mcl{M} \in \sys(\gtc{G_j})$ such that $\mcl{M}^\complet \lredd{\rLbl \ci{c}{\mu}>\ci{\mu}{c}:\ell}$.
                    Hence, the thesis follows from \ih2.
            \end{itemize}

        \item Case $n = n'+1$.
            By assumption, $\gtc{G} \lredd{\beta_1} \gtc{G'_1}$ where $c$ is not the sender or recipient in $\beta_1$.
            Hence, $\gtc{G} = a {!} b \gtBraces{ \msg i<S_i> \gtc. \gtc{G_i} }_{i \in I}$ where $\gtc{G'_1} = \gtc{G_j}$ for some $j \in I$.
            The router for $c$ at $\gtc{G}$ is thus generated through \cref{li:mpstAPCP:rtrElseElse} of \Cref{alg:mpstAPCP:router}.
            It follows from \Cref{l:mpstAPCP:routerNodepEqual} that this router is equal to the router for $c$ at $\gtc{G'_1}$, but with one less pass through \cref{li:mpstAPCP:rtrElseElse}.
            Given the same implementation process for $c$ as in $\mcl{N}$, we can construct a completable $\mcl{M} \in \sys(\gtc{G'_1})$ such that $\mcl{M}^\complet \lredd{\rLbl \ci{c}{\mu}>\ci{\mu}{c}:\ell}$.
            Hence, the thesis follows from \ih1.
            \qedhere
    \end{itemize}
\end{proof}
\fi

The proof of soundness relies on \Cref{p:mpstAPCP:indep}: if different reductions are simultaneously enabled for a given process, then they do not exclude each other.
That is, the same process is reached no matter the order in which those reductions are executed.
We refer to simultaneously enabled reductions as \emph{independent reductions}.

\begin{restatable}[Independent Reductions]{proposition}{pMpstAPCPIndep}
    \label{p:mpstAPCP:indep}
    Suppose given a process $\Omega \vdash P \typInf \Gamma$ and reduction labels $\alpha$ and $\alpha'_1,\ldots,\alpha'_n$ ($n \geq 1$) where $\alpha \notin \{\alpha'_1,\ldots,\alpha'_n\}$ (cf.\ \ifappendix \Cref{d:mpstAPCP:APCP:procLred}\else \Cref{d:APCP:procLred} on \Cpageref{d:APCP:procLred}\fi).
    If $P \lredd{\alpha}$ and $P \lredd{\alpha'_1} \ldots \lredd{\alpha'_n}$, then there exists a process $Q$ such that $P \lredd{\alpha} \lredd{\alpha'_1} \ldots \lredd{\alpha'_n} Q$ and $P \lredd{\alpha'_1} \ldots \lredd{\alpha'_n} \lredd{\alpha} Q$.
\end{restatable}

\ifmpstAPCP
\begin{sketch}
    By induction on $n$.
    The crux is that, since reductions $\alpha$ and $\alpha'_i$ are enabled simultaneously, they must occur in parallel parts of $P$.
    Hence, these reductions do not influence each other, proving the thesis.
\end{sketch}
\else
\begin{proof}
    \label{proof:mpstAPCP:indep}
    By induction on $n$:
    \begin{itemize}
        \item $n = 1$.
            By assumption, $P \lredd{\alpha}$ and $P \lredd{\alpha'_1}$.
            The proof proceeds by considering all possible combinations of shapes for $\alpha$ and $\alpha'_1$ (forwarder, send/receive, and selection/branching).

            Consider the case where ${\alpha = \rLbl x>y:a}$ and $\alpha'_1 = \rLbl w>z:b$.
            Because $P$ is well-typed, we infer that there are evaluation contexts $\evalCtx{E_1}$ and $\evalCtx{E_2}$ such that
            \[
                P \equiv \evalCtx{E_1}[\pRes{xy} ( \pOut x[a,c] \| \pIn y(a,c) ; P_1 )] \equiv \evalCtx{E_2}[\pRes{wz} ( \pOut w[b,d] \| \pIn z(b,d) ; P_2])
            \]
            \ifappendix (\Cref{d:mpstAPCP:APCP:evalCtx})\else (cf.\ \Cref{d:APCP:evalCtx} on \Cpageref{d:APCP:evalCtx})\fi.
            Since the reductions labeled $\alpha$ and $\alpha'_1$ are both enabled in $P$, it cannot be the case that \mbox{$x,y \in \fn(P_2)$} and $w,z \in \fn(P_1)$.
            Hence, there exists an evaluation context $\evalCtx{E_3}$ such that $P \equiv \evalCtx{E_3}[\pRes{xy} ( \pOut x[a,c] \| \pIn y(a,c) ; P_1 ) \| \pRes{wz} ( \pOut w[b,d] \| \pIn z(b,d) ; P_2 )]$.
            Then
            \[
                P \lredd{\alpha} Q_1 \equiv \evalCtx{E_3}[P_1 \| \pRes{wz} ( \pOut w[b,d] \| \pIn z(b,d) ; P_2) ]
            \]
            and $P \lredd{\alpha'_1} Q_2 \equiv \evalCtx{E_3}[\pRes{xy} ( \pOut x[a,c] \| \pIn y(a,c) ; P_1 ) \| P_2]$.
            Let \mbox{$Q = \evalCtx{E_3}[P_1 \| P_2]$}; then $Q_1 \lredd{\alpha'_1} Q$ and $Q_2 \lredd{\alpha} Q$.
            Hence, $P \lredd{\alpha} \lredd{\alpha'_1} Q$ and $P \lredd{\alpha'_1} \lredd{\alpha} Q$.

			All other cases proceed similarly.
            Note that when one of the reductions (say, $\alpha$) has a selection/branching label, such a reduction would discard some branches and thus possible behaviors.
            This is not an issue for establishing the thesis, because typability guarantees that the sub-process that enables the $\alpha'$-labeled reduction does not appear under the to-be-discarded branches. Hence, the execution of $\alpha$ will not jeopardize the $\alpha'$-labeled reduction.

        \item $n = n'+1$ for $n' \geq 1$.
            By the IH, ${P \lredd{\alpha} \lredd{\alpha'_1} \ldots \lredd{\alpha'_{n'}} Q'_1}$ and $P \lredd{\alpha'_1} \ldots \lredd{\alpha'_{n'}} P' \lredd{\alpha} Q'_1$.
            By assumption, $P' \lredd{\alpha'_n} Q'_2$.
            Since $P$ is well-typed, by \ifappendix \Cref{t:mpstAPCP:APCP:tp} (type preservation)\else \APCP's type preservation (\Cref{t:APCP:sr} on \Cpageref{t:APCP:sr})\fi, $P'$ is well-typed.
            Since $P' \lredd{\alpha} Q'_1$ and $P' \lredd{\alpha'_n} Q'_2$, we can follow the same argumentation as in the base case to show that $Q'_1 \lredd{\alpha'_n} Q$ and $Q'_2 \lredd{\alpha} Q$.
            Hence, $P \lredd{\alpha} \lredd{\alpha'_1} \ldots \lredd{\alpha'_{n'}} Q'_1 \lredd{\alpha'_n} Q$ and $P \lredd{\alpha'_1} \ldots \lredd{\alpha'_{n'}} P' \lredd{\alpha'_n} Q'_2 \lredd{\alpha} Q$.
            \qedhere
    \end{itemize}
\end{proof}
\fi

To understand the role of independent reductions in soundness, consider the following example.
We first introduce some notation which we also use in the proof of soundness: given an ordered sequence of reduction labels $\vec{\alpha} = (\alpha_1,\ldots,\alpha_k)$, we write $P \lredd{\vec{\alpha}} Q$ to denote $P \lredd{\alpha_1} \ldots \lredd{\alpha_k} Q$.

\begin{example}
    \label{x:mpstAPCP:presound}
    The global type $\gtc{G} = \mu X \gtc. a {!} b (\msg 1<S>) \gtc. c {!} d (\msg 1<S>) \gtc. X$
    features two independent exchanges.
    Consider a network $\mcl{N} \in \sys(\gtc{G})$.
    Let $\vec{\alpha}$ denote the sequence of labeled reductions necessary to complete the exchange in $\gtc{G}$ between $a$ and $b$, and $\vec{\beta}$ similarly for the exchange between $c$ and $d$.
    Assuming that communication with routers is not blocked by implementation processes, we have $\mcl{N}^\complet \lredd{\vec{\alpha}}$ and $\mcl{N}^\complet \lredd{\vec{\beta}}$, because the exchanges are independent.

    Now, suppose that from $\mcl{N}^\complet$ we observe $m$ times the sequence of $\vec{\beta}$ reductions: $\mcl{N}^\complet \underbrace{\lredd{\vec{\beta}} \ldots \lredd{\vec{\beta}}}_{\text{$m$ times}} N'$.
    We see that $N'$ is not a network of a global type reachable from $\gtc{G}$: there are still $m$ exchanges between $a$~and~$b$ pending.
    Still, we can exhibit a series of transitions from $\gtc{G}$ that includes $m$ times the exchange between $c$ and $d$:
    \[
        \gtc{G} \underbrace{\lredd{a {!} b (\msg 1<S>)} \lredd{c {!} d (\msg 1<S>)} \ldots \lredd{a {!} b (\msg 1<S>)} \lredd{c {!} d (\msg 1<S>)}}_{\text{$m$ times}} \gtc{G}
    \]
    Following these transitions, we can exhibit a corresponding sequence of reductions from~$\mcl{N}^\complet$ that includes $m$ times the sequence $\vec{\beta}$ and ends up in another network \mbox{$\mcl{M} \in \sys(\gtc{G})$}:
    \[
        \mcl{N}^\complet \underbrace{\lredd{\vec{\alpha}} \lredd{\vec{\beta}} \ldots \lredd{\vec{\alpha}} \lredd{\vec{\beta}}}_{\text{$m$ times}} \mcl{M}^\complet
    \]
    At this point it is crucial that from $\mcl{N}^\complet$ the sequences of reductions $\vec{\alpha}$ and $\vec{\beta}$ can be performed independently.
    Hence, by \Cref{p:mpstAPCP:indep}, $\mcl{N}^\complet \underbrace{\lredd{\vec{\beta}} \ldots \lredd{\vec{\beta}}}_{\text{$m$ times}} N' \underbrace{\lredd{\vec{\alpha}} \ldots \lredd{\vec{\alpha}}}_{\text{$m$ times}} \mcl{M}^\complet$.
\end{example}

\ifmpstAPCP\else
In the proof of soundness, whenever we assure that certain reductions are independent, we refer to those assurances as \emph{independence facts} (\textbf{IFacts}).
Also, in the proof we consider labeled reductions, and distinguish between \emph{protocol} and \emph{implementation} reductions: the former are reductions with labels that indicate any interaction with a router, and the latter are any other reductions (which, by the definition of networks, can only occur within participant implementation processes).
By a slight abuse of notation, given ordered sequences of reduction labels $\vec{\alpha}$ and $\vec{\alpha}'$, we write $\vec{\alpha}' \subseteq \vec{\alpha}$ to denote that $\vec{\alpha}'$ is a subsequence of $\vec{\alpha}$, where the labels in $\vec{\alpha}'$ appear in the same order in $\vec{\alpha}$ but not necessarily in sequence (and similarly for~\mbox{$\vec{\alpha}' \subset \vec{\alpha}$}).
With $\vec{\alpha} \setminus \vec{\alpha}'$ we denote the sequence obtained from $\vec{\alpha}$ by removing all the labels in~$\vec{\alpha}'$, and $\vec{\alpha} \cup \vec{\alpha}'$ denotes the sequence obtained by adding the labels from $\vec{\alpha}'$ to the end of $\vec{\alpha}$.
\fi

\begin{restatable}[Operational Correspondence: Soundness]{theorem}{tMpstAPCPSoundness}
    \label{t:mpstAPCP:soundness}
    Suppose given a relative well-formed global type $\gtc{G}$ and a completable $\mcl{N} \in \sys(\gtc{G})$.
    For every ordered sequence of $k \geq 0$ reduction labels $\vec{\alpha} = (\alpha_1, \ldots, \alpha_k)$ and $N'$ such that $\mcl{N}^\complet \lredd{\vec{\alpha}} N'$, there exist $\gtc{G'}$ and
    $\beta_1,\ldots,\beta_n$ (with $n \geq 0$) such that (i)~$\gtc{G} \lredd{\beta_1} \ldots \lredd{\beta_n} \gtc{G'}$
    and (ii)~$N' \redd^\ast \mcl{M}^\complet$, with $\mcl{M} \in \sys(\gtc{G'})$.
\end{restatable}

\ifmpstAPCP
\begin{sketch}
    By induction on the structure of $\gtc{G}$.
    Then we apply induction on the size of $\vec{\alpha}$.
    We follow the reductions in $\vec{\alpha}$, but in the order as they appear in $\gtc{G}$.
    Eventually, we reach a $\gtc{G'}$ that is reachable from $\gtc{G}$ (cf.\ \Cref{d:mpstAPCP:globalLts}), as well as a network of $\gtc{G'}$ that is reachable from $N'$.
    We then apply \Cref{p:mpstAPCP:indep} show that the same network can be reached by following reductions in the order as they appear in $\vec{\alpha}$.
\end{sketch}
\else
\begin{proof}
    \label{proof:mpstAPCP:soundness}
    By induction on the structure of $\gtc{G}$; we detail the interesting cases of labeled exchanges with implicitly unfolded recursive definitions.
    We exhibit transitions $\gtc{G} \lredd{\beta_1} \ldots \lredd{\beta_n} \gtc{G'}$ and establish a corresponding sequence of reductions $\mcl{N}^\complet \redd^\ast \mcl{M}^\complet$ that includes all the labels in~$\vec{\alpha}$, with $\mcl{M} \in \sys(\gtc{G'})$.
    During this step, we assure the independence between the observed reductions $\vec{\alpha}$ and the reductions we establish (\textbf{IFacts}).
    Using these independence assurances, we show that also $N' \redd^\ast \mcl{M}^\complet$.

    We apply induction on the size of $\vec{\alpha}$ (\ih1) to show the existence of (i) $\gtc{G'}$ and $\beta_1,\ldots,\beta_n$ such that (i) $\gtc{G} \lredd{\beta_1} \ldots \lredd{\beta_n} \gtc{G'}$ and (ii) $\mcl{N}^\complet \redd^\ast \mcl{M}^\complet$ including all reductions in $\vec{\alpha}$, with $\mcl{M} \in \sys(\gtc{G'})$:
    \begin{itemize}
        \item
            Base case: then $\vec{\alpha}$ is empty, and the thesis holds trivially, with $\gtc{G'} = \gtc{G}$ and $\mcl{M} = \mcl{N}$.

        \item
            Inductive case: then $\vec{\alpha}$ is non-empty.

            By the definition of networks (\Cref{d:mpstAPCP:networks}), we know that reductions starting at~$\mcl{N}^\complet$ are protocol reductions related to an independent exchange in $\gtc{G}$, or implementation reductions.
            Every protocol reduction in $\vec{\alpha}$ is related to some exchange in $\gtc{G}$, and so we can group sequences of protocol reductions related to the same exchange.
            By construction, every such sequence of protocol reductions $\vec{\alpha}_\ast \subseteq \vec{\alpha}$ starts with an implementation sending a label to a router, i.e., with a label of the form $\alpha_\ast = \rLbl \ci{c}{\mu}>\ci{\mu}{c}:\ell$.
            For each such $\alpha_\ast$, the router in $\mcl{N}$ of the sender $c$ has been synthesized from $\gtc{G}$ in a finite number of inductive steps.
            We take the $\alpha_\ast$ that originates from the router synthesized in the least number of steps.
            This gives us the $\vec{\alpha}_\ast$ starting with $\alpha_\ast$ that relates to an exchange in $\gtc{G}$ which is not prefixed by exchanges relating to any of the other $\vec{\alpha}'_\ast \subseteq \vec{\alpha} \setminus \vec{\alpha}_\ast$.

            Networks are well-typed by definition.
            None of the reductions in $\vec{\alpha}_\ast$ are blocked by protocol reductions appearing earlier in $\vec{\alpha}$ (\textbf{IFact~1}): they originate from exchanges in $\gtc{G}$ appearing after the exchange related to $\vec{\alpha}_\ast$, and the priorities in their related types are thus higher than those in the types related to $\vec{\alpha}_\ast$, i.e., blocking by receiving or branching would contradict the well-typedness of $\mcl{N}^\complet$.
            However, it may be that some implementation reductions $\vec{\alpha}_+ \subseteq \vec{\alpha} \setminus \vec{\alpha}_\ast$ do block the reductions in $\vec{\alpha}_\ast$; they are also not blocked by any prior protocol reductions due to priorities (\textbf{IFact~2}).
            Hence, from $\mcl{N}^\complet$ we can perform the implementation reductions in $\vec{\alpha}_+$.
            By \ifappendix type preservation (\Cref{t:mpstAPCP:APCP:tp})\else \APCP's type preservation (\Cref{t:APCP:sr} on \CPageref{t:APCP:sr})\fi, this results in another completed network $\mcl{N}_0^\complet$ of $\gtc{G}$.
            This establishes the reduction sequence $\mcl{N}^\complet \lredd{\vec{\alpha}_+} \mcl{N}_0^\complet \lredd{\vec{\alpha}_\ast}$.

            By \Cref{l:mpstAPCP:delayedPrefix}, there are $m \geq 0$ transitions $\gtc{G} \lredd{\beta_1} \gtc{G_1} \ldots \lredd{\beta_m} \gtc{G_m}$ where the initial prefix of $\gtc{G_m}$ corresponds to the labeled choice by the implementation of $c$: $\gtc{G_m} = c {!} d \gtBraces{ \msg i<S_i> \gtc. \gtc{G'_i} }_{i \in I}$, with $\ell \in I$.
            Additionally, $\gtc{G_m}$ contains exchanges related to every sequence of protocol reductions in $\vec{\alpha} \setminus \vec{\alpha}_+ \setminus \vec{\alpha}_\ast$: all these sequences start with a selection from implementation to router, and thus the involved participants do not depend on any of the exchanges between $\gtc{G}$ and $\gtc{G_m}$, such that \Cref{l:mpstAPCP:routerNodepEqual} applies.
            To establish a sequence of reductions from $\mcl{N}_0^\complet$ to the completion of a network $\mcl{N}_m \in \sys(\gtc{G_m})$, we
            apply induction on $m$ (\ih2):
            \begin{itemize}
                \item
                    The base case where $m = 0$ is trivial, with $\gtc{G_m} = \gtc{G}$ and thus $\mcl{N}_m^\complet = \mcl{N}_0^\complet$.

                \item
                    In the inductive case, following the same approach as in the proof of completeness (\Cref{t:mpstAPCP:completeness}), we reduce $\mcl{N}_0^\complet \redd^\ast \mcl{N}_1^\complet$ such that $\mcl{N}_1 \in \sys(\gtc{G_1})$.
                    Then, by~\ih2, $\mcl{N}_1^\complet \redd^\ast \mcl{N}_m^\complet$ where $\mcl{N}_m \in \sys(\gtc{G_m})$.
                    Note that these reductions may require implementation reductions to unblock protocol reductions, and these implementation reductions may appear in $\vec{\alpha}$.
                    None of the reductions from $\mcl{N}_0^\complet$ to $\mcl{N}_m^\complet$ can be blocked by any of the other protocol reductions in $\vec{\alpha}$, following again from priorities in types; hence, the leftover reductions in $\vec{\alpha}$ are independent from these reductions (\textbf{IFact~3}).
                    Additionally, the sequence of protocol reductions $\vec{\alpha}_\ast$ was already enabled from $\mcl{N}_0^\complet$, so those reductions are also independent (\textbf{IFact~4}).
            \end{itemize}

            We know that $\mcl{N}_m^\complet \lredd{\rLbl \ci{c}{\mu}>\ci{\mu}{c}:\ell}$ and $\gtc{G_m} \lredd{c {!} d (\msg \ell<S_\ell>)} \gtc{G'_\ell}$.
            From $\mcl{N}_m^\complet$, we again follow the proof of completeness to show that $\mcl{N}_m^\complet \redd^\ast \mcl{M}_\ell^\complet$, where $\mcl{M}_\ell \in \sys(\gtc{G'_\ell})$.
            Given the definition of routers, it must be that all the reductions in $\vec{\alpha}_\ast$ appear in this sequence of reductions.
            Let $\vec{\alpha}' \subset \vec{\alpha}$ denote the leftover reductions from $\vec{\alpha}$ (i.e., $\vec{\alpha}$ except all reductions that occurred between $\mcl{N}^\complet$ and $\mcl{M}_\ell^\complet$, including $\vec{\alpha}_\ast$ and $\vec{\alpha}_+$).
            By \mbox{\textbf{IFacts~1--4}}, $\mcl{M}_\ell^\complet \lredd{\vec{\alpha}'} M'$.
            Then by \ih1, there exist $\gtc{G'}$ and $\beta_{m+2},\ldots,\beta_n$ (with $n \geq m+1$) such that (i)~\mbox{$\gtc{G'_\ell} \lredd{\beta_{m+2}} \ldots \lredd{\beta_n} \gtc{G'}$} and (ii)~$\mcl{M}_\ell^\complet \redd^\ast \mcl{M}^\complet$ including all reductions in~$\vec{\alpha}'$, with \mbox{$\mcl{M} \in \sys(\gtc{G'})$}.
            Let \mbox{$\beta_{m+1} = c {!} d (\msg \ell<S_\ell>)$}.
            We have shown the existence of $\gtc{G'}$ and $\beta_1,\ldots,\beta_n$ such that (i)~$\gtc{G} \lredd{\beta_1} \ldots \lredd{\beta_m} \gtc{G_m} \lredd{\beta_{m+1}} \gtc{G'_\ell} \lredd{\beta_{m+2}} \ldots \lredd{\beta_n} \gtc{G'}$ and (ii)~\mbox{$\mcl{N}^\complet \redd^\ast \mcl{N}_m^\complet \redd^\ast \mcl{M}_\ell^\complet \redd^\ast \mcl{M}^\complet$} including all reductions in $\vec{\alpha}$, with $\mcl{M} \in \sys(\gtc{G'})$.
    \end{itemize}

    We are left to show that from $\mcl{N}^\complet \redd^\ast \mcl{M}^\complet$ and $\mcl{N}^\complet \lredd{A} N'$, we can conclude that $N' \redd^\ast \mcl{M}^\complet$.
    We apply induction on the size of $\vec{\alpha}$ (\ih3), using \textbf{IFacts~1--4} and \Cref{p:mpstAPCP:indep}:
    \begin{itemize}
        \item
            Base case: Then $\vec{\alpha}$ is empty, there is nothing to do, and the thesis is proven.

        \item
            Inductive case: Then $\vec{\alpha} = \vec{\alpha}' \cup (\alpha')$.
            By \ih3, $\mcl{N}^\complet \lredd{\vec{\alpha}'} N'' \redd^\ast N''' \lredd{\alpha'} \redd^\ast \mcl{M}^\complet$.
            Moreover, by assumption, $N'' \lredd{\alpha'} N'$.
            \textbf{IFacts~1--4} show that the $\alpha'$-labeled reduction is independent from the reductions between $N''$ and $N'''$.
            Hence, by \Cref{p:mpstAPCP:indep}, we have $\mcl{N}^\complet \lredd{\vec{\alpha}'} N'' \lredd{\alpha'} N' \redd^\ast \mcl{M}^\complet$.
            That is, $\mcl{N}^\complet \lredd{\vec{\alpha}} N' \redd^\ast \mcl{M}^\complet$, proving the thesis.
            \qedhere
    \end{itemize}
\end{proof}
\fi

In the light of \Cref{t:mpstAPCP:soundness}, let us revisit \Cref{x:mpstAPCP:presound}:

\begin{example}[Revisiting \Cref{x:mpstAPCP:presound}]
    Recall the global type 
    \[
        \gtc{G} = \mu X \gtc. a {!} b (\msg 1<S> \gtc. c {!} d (\msg 1<S>) \gtc. X
    \]
    from \Cref{x:mpstAPCP:presound}, with two independent exchanges.
    
    We take some $\mcl{N} \in \sys(\gtc{G})$ such that $\mcl{N}^\complet \underbrace{\lredd{\vec{\beta}} \ldots \lredd{\vec{\beta}}}_{\text{$m$ times}} N'$, where $\vec{\beta}$ denotes the sequence of reduction labels corresponding to the exchange between $c$ and $d$.
    By \Cref{t:mpstAPCP:soundness}, there indeed are $\gtc{G'}$ and $\beta_1,\ldots,\beta_n$ such that $\gtc{G} \lredd{\beta_1} \ldots \lredd{\beta_n} \gtc{G'}$ and $N' \redd^\ast \mcl{M}^\complet$, with $\mcl{M} \in \sys(\gtc{G'})$.
    To be precise, following \Cref{t:mpstAPCP:soundness}, indeed
    \[
        \gtc{G} \underbrace{\lredd{a {!} b (\msg 1<S>)} \lredd{c {!} d (\msg 1<S>)} \ldots \lredd{a {!} b (\msg 1<S>)} \lredd{c {!} d (\msg 1<S>)}}_{\text{$m$ times}} \gtc{G} \quad \text{ and } \quad \mcl{N}^\complet \underbrace{\lredd{\vec{\beta}} \ldots \lredd{\vec{\beta}}}_{\text{$m$ times}} N' \underbrace{\lredd{\vec{\alpha}} \ldots \lredd{\vec{\alpha}}}_{\text{$m$ times}} \mcl{M}^\complet
    \]
    where $\vec{\alpha}$ is the sequence of reduction labels corresponding to the  exchange between $a$ and $b$ and ${\mcl{M} \in \sys(\gtc{G})}$.
\end{example}

\subsection{Routers Strictly Generalize Centralized Orchestrators}
\label{s:mpstAPCP:analysis:mediums}

Unlike our decentralized analysis, previous analyses of global types using binary session types rely on centralized orchestrators (called mediums~\cite{conf/forte/CairesP16} or arbiters~\cite{conf/concur/CarboneLMSW16}).
Here, we show that our approach strictly generalizes these centralized approaches.
Readers interested in our decentralized approach in action may safely skip this section and go directly to \Cref{s:mpstAPCP:examples}.

We introduce an algorithm that synthesizes an orchestrator---a single process that orchestrates the interactions between a protocol's participants (\Cref{s:mpstAPCP:analysis:mediums:mediumSynth}).
We show that the composition of this orchestrator with a context of participant implementations is behaviorally equivalent to the specific case in which routed implementations are organized in a \emph{centralized composition} (\Cref{t:mpstAPCP:mediumBisim} in \Cref{s:mpstAPCP:analysis:mediums:bisim}).

\subsubsection{Synthesis of Orchestrators}
\label{s:mpstAPCP:analysis:mediums:mediumSynth}

\begin{algorithm}[t]
    \DontPrintSemicolon
    \SetAlgoNoEnd
    \SetInd{.2em}{.7em}
    \SetSideCommentRight
    \Def{$\mdm{\tilde{q}}{\gtc{G}}$}{
        \Switch{$\gtc{G}$}{
            \uCase{$s {!} r \gtBraces{ \msg i<S_i> \gtc. \gtc{G_i} }_{i \in I}$}{ \label{li:mpstAPCP:mdmComm}
                $\deps \deq \{ q \in \tilde{q} \mid \hdep(q,s,\gtc{G}) \vee \hdep(q,r,\gtc{G}) \}$ \; \label{li:mpstAPCP:mdmDep}
                \KwRet \label{li:mpstAPCP:mdmCommRet}
                $\pBra* \ci{\mu}{s} > {\{ i : \pSel* \ci{\mu}{r} < i \cdot \underline{{(\pSel* \ci{\mu}{q} < i)}_{q \in \deps}} \cdot \pIn \ci{\mu}{s}(v) ; \pOut* \ci{\mu}{r}[w] \cdot (\pFwd [v<>w] \| \mdm{\tilde{q}}{\gtc{G_i}}) \}_{i \in I}}$
            }

            \uCase{$\mu X \gtc. \gtc{G'}$}{ \label{li:mpstAPCP:mdmRecDef}
                $\tilde{q}' \deq \{ q \in \tilde{q} \mid \gtc{G} \locProjD^0 q \neq \bullet \}$ \; \label{li:mpstAPCP:mdmRecComp}
                \lIf{$\tilde{q}' \neq \emptyset$}{ \label{li:mpstAPCP:mdmRecRet}
                    \KwRet $\pRec X({(\ci{\mu}{q})}_{q \in \tilde{q}'}) ; \mdm{\tilde{q}'}{\gtc{G'}}$
                }
                \lElse{ \label{li:mpstAPCP:mdmRecInact}
                    \KwRet $\0$
                }
            }

            \lCase{$X$}{ \label{li:mpstAPCP:mdmRecCall}
                \KwRet $\pCall X<{(\ci{\mu}{q})}_{q \in \tilde{q}}>$
            }

            \lCase{$\tSkip \gtc. \gtc{G'}$}{ \label{li:mpstAPCP:mdmSkip}
                \KwRet $\mdm{\tilde{q}}{\gtc{G'}}$
            }

            \lCase{$\tEnd$}{\label{li:mpstAPCP:mdmEnd}
                \KwRet $\0$
            }
        }
    }
    \caption{Synthesis of Orchestrator Processes (\Cref{d:mpstAPCP:medium}).}
    \label{alg:mpstAPCP:medium}
\end{algorithm}

We define the synthesis of an orchestrator from a global type.
The orchestrator of $\gtc{G}$ will have a channel endpoint $\ci{\mu}{p_i}$ for connecting to the process implementation of every $p_i \in \prt(\gtc{G})$.

\begin{definition}[Orchestrator]
    \label{d:mpstAPCP:medium}
    Given a global type $\gtc{G}$ and participants $\tilde{q}$, \Cref{alg:mpstAPCP:medium} defines the synthesis of an \emph{orchestrator process}, denoted $\mdm{\tilde{q}}{\gtc{G}}$, that orchestrates interactions according to $\gtc{G}$.
\end{definition}

\Cref{alg:mpstAPCP:medium} follows a similar structure as the router synthesis algorithm (\Cref{alg:mpstAPCP:router}).
The input parameter $\tilde{q}$ keeps track of active participants, making sure recursions are well-defined; it should be initialized as $\prt(\gtc{G})$.

We briefly discuss how the orchestrator process is generated.
The interesting case is an exchange $p {!} q \gtBraces{ \msg i<U_i> \gtc. \gtc{G_i} }_{i \in I}$ (\cref{li:mpstAPCP:mdmComm}), where the algorithm combines the several cases of the router's algorithm (that depend on the involvement of the router's participant).
First, the sets of participants $\deps$ that depend on the sender and on the recipient are computed (\cref{li:mpstAPCP:mdmDep}) using the auxiliary predicate $\hdep$~(cf.\ \Cref{d:mpstAPCP:hdep}).
Then, the algorithm returns a process (\cref{li:mpstAPCP:mdmCommRet}) that receives a label $i \in I$ over $\ci{\mu}{s}$; forwards it over $\ci{\mu}{r}$ and over $\ci{\mu}{q}$ for all $q \in \deps$; receives a channel over $\ci{\mu}{s}$; forwards it over $\ci{\mu}{r}$; and continues as $\mdm{\tilde{q}}{\gtc{G_i}}$.

The synthesis of a recursive definition $\mu X \gtc. \gtc{G'}$ (\cref{li:mpstAPCP:mdmRecDef}) requires care, as the set of active participants $\tilde{q}$ may change.
In order to decide which $q \in \tilde{q}$ are active in $\gtc{G'}$, the algorithm computes the local projection of $G$ onto each $q \in \tilde{q}$ to determine the orchestrator's future behavior on $\ci{\mu}{q}$, creating a new set $\tilde{q}'$ with those $q \in \tilde{q}$ for which the projection is different from $\bullet$ (\cref{li:mpstAPCP:mdmRecComp}).
Then, the algorithm returns a recursive process with as context the channel endpoints $\ci{\mu}{q}$ for $q \in \tilde{q}'$, with $\mdm{\tilde{q}'}{\gtc{G'}}$ as the body.

The synthesis of a recursive call $X$ (\cref{li:mpstAPCP:mdmRecCall}) yields a recursive call with as context the channels $\ci{\mu}{q}$ for $q \in \tilde{q}$.
Finally, for $\tSkip \gtc. \gtc{G'}$ (\cref{li:mpstAPCP:mdmSkip}) the algorithm returns the orchestrator for~$\gtc{G'}$, and for $\tEnd$ (\cref{li:mpstAPCP:mdmEnd}) the algorithm returns $\0$.

There is a minor difference between the orchestrators synthesized by \Cref{alg:mpstAPCP:medium} and the mediums defined by Caires and P\'{e}rez~\cite{conf/forte/CairesP16}.
The difference is in the underlined portion in \cref{li:mpstAPCP:mdmCommRet}, which denotes explicit messages (obtained via dependency detection) needed to deal with non-local choices.
The mediums by Caires and Pérez do not include such communications, as their typability is based on local types, which rely on a merge operation at projection time.
The explicit actions in \cref{li:mpstAPCP:mdmCommRet} make the orchestrator compatible with participant implementations that connect with routers.
Aside from these actions, our concept of orchestrator is essentially the same as that of mediums.

Crucially, orchestrators can be typed using local projection (cf.\ \Cref{d:mpstAPCP:locProj}) similar to the typing of routers using relative projection (cf.\ \Cref{t:mpstAPCP:routerTypes}).
This result follows by construction (see \ifappendix \Cref{as:mpstAPCP:analysis:medium:mediumSynth} \else \cite{thesis/vdHeuvel24ex} \fi for details):

\begin{restatable}{theorem}{tMpstAPCPMediumTypes}
    \label{t:mpstAPCP:mediumTypes}
    Given a closed, relative well-formed global type $\gtc{G}$,
    \[
        \emptyset \vdash \mdm{\prt(\gtc{G})}{\gtc{G}} \typInf {(\ci{\mu}{p}: \ol{(\gtc{G} \locProjD^0 p)})}_{p \in \prt(\gtc{G})}.
    \]
\end{restatable}

\ifmpstAPCP
\else
\begin{proof}
    We prove a more general statement.
    Suppose given a closed, relative well-formed global type $\gtc{G}$.
    Also, suppose given a global type $\gtc{G_s} \gSub{\globCtx{C}} \gtc{G}$.
    Consider:
    \begin{itemize}
        \item
            the participants that are active in $\gtc{G_s}$:
            \[
                \tilde{q} \deq \{ q \in \prt(\gtc{G}) \mid \exists p \in \prt(\gtc{G}).~ (p,q) \in \act(\globCtx{C},\gtc{G}) \};
            \]

        \item
            the absolute priority of $\gtc{G_s}$:
            $\pri_{\globCtx{C}} \deq \ctxPri(\globCtx{C})$;

        \item
            the sequence of bound recursion variables of $\gtc{G_s}$:
            $\widetilde{X_{\globCtx{C}}} = \ctxBind(\globCtx{C})$;

        \item
            for every $X \in \widetilde{X_{\globCtx{C}}}$:
            \begin{itemize}
                \item
                    the body of the recursive definition on $X$ in $\gtc{G}$:
                    $\gtc{G}_X \deq \recDef(X,\gtc{G})$;

                \item
                    the participants that are active in $\gtc{G}_X$:
                    \[
                        \tilde{q}_X \deq \{ q \in \prt(\gtc{G}) \mid \exists p \in \prt(\gtc{G}).~ (p,q) \in \recAct(X,\gtc{G}) \};
                    \]

                \item
                    the absolute priority of $\gtc{G}_X$:
                    $\pri_X \deq \varPri(X,\gtc{G})$;

                \item
                    the sequence of bound recursion variables of $\gtc{G}_X$ excluding $X$:
                    \[
                        \widetilde{Y_X} \deq \subBind(\mu X \gtc. \gtc{G}_X,\gtc{G});
                    \]

                \item
                    the type required for $\ci{\mu}{q}$ for a recursive call on $X$:
                    \[
                        A_{X,q} \deq \deepUnfold(\ol{\gtc{G}_X \locProjD^{\pri_X} q}, {(Y,t_Y,\ol{\gtc{G}_Y \locProjD^{\pri_Y} q})}_{Y \in \widetilde{Y_X}});
                    \]

                \item
                    the minimum lift for typing a recursive definition on $X$:
                    \[
                        t_X \deq \maxpr \left( {(A_{X,q})}_{q \in \tilde{q}_X} \right) + 1;
                    \]
            \end{itemize}

        \item
            the type expected for $\ci{\mu}{q}$ for the orchestrator for $\gtc{G_s}$:
            \[
                D_q \deq \deepUnfold(\ol{\gtc{G_s} \locProjD^{\pri_{\globCtx{C}}} q}, {(X,t_X,\ol{\gtc{G}_X \locProjD^{\pri_X} q})}_{X \in \widetilde{X_{\globCtx{C}}}}).
            \]
    \end{itemize}
    Then, we have:
    \[
        {\left( X: {\big(A_{X,q}\big)}_{q \in \tilde{q}_X} \right)}_{X \in \widetilde{X_{\globCtx{C}}}}
        \vdash \mdm{\tilde{q}}{\gtc{G_s}} \typInf
        {\left( \ci{\mu}{q}: D_q \right)}_{q \in \tilde{q}}.
    \]
    Similar to how \Cref{t:mpstAPCP:routerTypes} follows from \Cref{t:mpstAPCP:routerTypesGen}, the thesis follows as a corollary from this more general statement (cf.\ the proof of \Cref{t:mpstAPCP:routerTypes} on \Cpageref{proof:mpstAPCP:routerTypes}).

    We apply induction on the structure of $\gtc{G_s}$, with six cases as in \Cref{alg:mpstAPCP:medium}.
    We only detail the cases of exchange and recursion.
    \begin{itemize}
        \item
            \emph{Exchange}: $\gtc{G_s} = s {!} r \gtBraces{ \msg i<S_i> \gtc. \gtc{G_i} }_{i \in I}$ (\cref{li:mpstAPCP:mdmComm}).

            Following similar reasoning as in the case for exchange in the proof of \Cref{t:mpstAPCP:routerTypesGen}, we can omit the unfoldings on types, as well as the recursive context.

            Let $\deps_s \deq \{ q \in \tilde{q} \mid \hdep(q,s,\gtc{G_s}) \}$ and $\deps_r \deq \{ q \in \tilde{q} \setminus \deps_s \mid \hdep(q,r,\gtc{G_s}) \}$.
            Note that $\deps_s \cup \deps_r$ coincides with $\deps$ as defined on \cref{li:mpstAPCP:mdmDep} and that $s,r \notin \deps_s \cup \deps_r$.

            Let us  take stock of  the types we expect for each of the orchestrator's channels.
            \begin{align}
                & \text{For $\ci{\mu}{s}$ we expect}
                &
                \ol{\gtc{G_s} \locProjD^\pri s}
                &= \ol{\oplus^\pri \{ i : \msgProp{S_i} \tensor^{\pri+1} (\gtc{G_i} \locProjD^{\pri+4} s) \}_{i \in I}}
                \nonumber \\
                &
                &
                &= \&^\pri \{ i : \ol{\msgProp{S_i}} \parr^{\pri+1} \ol{(\gtc{G_i} \locProjD^{\pri+4} s)} \}_{i \in I}.
                \label{eq:mpstAPCP:mSType}
                \\
                & \text{For $\ci{\mu}{r}$ we expect}
                &
                \ol{\gtc{G_s} \locProjD^\pri r}
                &= \ol{\&^{\pri+2} \{ i : \ol{\msgProp{S_i}} \parr^{\pri+3} (\gtc{G_i} \locProjD^{\pri+4} r) \}_{i \in I}}
                \nonumber \\
                &
                &
                &= \oplus^{\pri+2} \{ i : \msgProp{S_i} \tensor^{\pri+3} \ol{(\gtc{G_i} \locProjD^{\pri+4} r)} \}_{i \in I}.
                \label{eq:mpstAPCP:mRType}
                \\
                & \text{For each $q \in \deps_s$,}
                \nonumber \\
                & \text{for $\ci{\mu}{q}$ we expect}
                &
                \ol{\gtc{G_s} \locProjD^\pri q}
                &= \ol{\&^{\pri+2} \{ i : (\gtc{G_i} \locProjD^{\pri+4} q) \}_{i \in I}}
                \nonumber \\
                &
                &
                &= \oplus^{\pri+2} \{ i : \ol{(\gtc{G_i} \locProjD^{\pri+4} q)} \}_{i \in I}.
                \label{eq:mpstAPCP:mDepSType}
                \\
                & \text{For each $q \in \deps_r$,}
                \nonumber \\
                & \text{for $\ci{\mu}{q}$ we expect}
                &
                \ol{\gtc{G_s} \locProjD^\pri q}
                &= \ol{\&^{\pri+3} \{ i : (\gtc{G_i} \locProjD^{\pri+4} q) \}_{i \in I}}
                \nonumber \\
                &
                &
                &= \oplus^{\pri+3} \{ i : \ol{(\gtc{G_i} \locProjD^{\pri+4} q)} \}_{i \in I}.
                \label{eq:mpstAPCP:mDepRType}
                \\
                & \text{For each $q \in \tilde{q} \setminus \deps_s \setminus \deps_r \setminus \{s,r\}$,}
                \nonumber \\
                & \text{for $\ci{\mu}{q}$ we expect}
                &
                \ol{\gtc{G_s} \locProjD^\pri q}
                &= \ol{\gtc{G_{i'}} \locProjD^{\pri+4} q} ~\text{for any $i' \in I$}.
                \label{eq:mpstAPCP:mDepOtherType}
            \end{align}

            Let us now consider the process returned by \Cref{alg:mpstAPCP:router}, with each prefix marked with a number.
            \[
                \mdm{\tilde{q}}{\gtc{G}} =
                \underbrace{
                    \pBra* \ci{\mu}{s} > \{ i :
                }_{1}
                \underbrace{
                    \pSel* \ci{\mu}{r} < i
                }_{2_i} \cdot
                \underbrace{
                    {(\pSel* \ci{\mu}{q} < i)}_{q \in \deps}
                }_{3_i} \cdot
                \underbrace{
                    \pIn \ci{\mu}{s}(v)
                }_{4_i} ;
                \underbrace{
                    \pOut* \ci{\mu}{r}[w]
                }_{5_i} \cdot
                (\pFwd [v<>w] \| \mdm{\tilde{q}}{\gtc{G_i}})
                \}_{i \in I}
            \]

            For each $i' \in I$, let $\globCtx{C_{i'}} \deq \globCtx{C}[s {!} r (\gtBraces{ \msg i<S_i> \gtc. \gtc{G_i} }_{i \in I \setminus \{i'\}} \cup \gtBraces{ \msg i'<S_{i'}> \gtc. \globHole })]$.
            Clearly, $\gtc{G_{i'}} \gSub{\globCtx{C_{i'}}} \gtc{G}$.
            Also, because we are not adding recursion binders, the current value of $\tilde{q}$ is appropriate for the IH.
            With $\globCtx{C_{i'}}$ and $\tilde{q}$, we apply the IH to obtain the typing of $\mdm{\tilde{q}}{\gtc{G_{i'}}}$, where priorities start at $\ctxPri(\globCtx{C_{i'}}) = \ctxPri(\globCtx{C}) + 4$ (cf.\ \Cref{d:mpstAPCP:ctxPri}).
            Following these typings, \Cref{f:mpstAPCP:tderiii} gives the typing of $\mdm{\tilde{q}}{\gtc{G_s}}$, referring to parts of the process by the number marking its foremost prefix above.

            Clearly, the priorities in the derivation in \Cref{f:mpstAPCP:tderiii} meet all requirements.
            The order of the applications of $\oplus^\star$ for each $q \in \deps_s \cup \deps_r$ does not matter, since the selections are asynchronous.

            \begin{figure}[t!]
                \ifthesis\else \begin{mdframed} \fi
                    \[
                        \def\defaultHypSeparation{\hskip3ex}
                        \def\ScoreOverhang{1pt}
                        \begin{bussproof}
                            \bussAx[\ruleLabel{typ-fwd}]{
                                \forall i \in I.~ \vdash \pFwd [v<>w] \typInf v:\ol{S_i}, w: S_i
                            }
                            \bussAssume{
                                \forall i \in I.~ \vdash \mdm{\tilde{q}}{\gtc{G_i}} \typInf {(\ci{\mu}{q}: \ol{\gtc{G_i} \locProjD^{\pri+4} q})}_{q \in \tilde{q}}
                            }
                            \bussBin[\ruleLabel{typ-par}]{
                                \forall i \in I.~ \vdash \pFwd [v<>w] \| \mdm{\tilde{q}}{\gtc{G_i}} \typInf v: \ol{S_i}, w: S_i, {(\ci{\mu}{q}: \ol{\gtc{G_i} \locProjD^{\pri+4}} q)}_{q \in \tilde{q}}
                            }
                            \bussUn[\ruleLabel{typ-send$\star$}]{
                                \forall i \in I.~ \vdash 5_i \typInf v: \ol{S_i}, \ci{\mu}{r}: \msgProp{S_i} \tensor^{\pri+3} \ol{(\gtc{G_i} \locProjD^{\pri+4} r)}, {(\ci{\mu}{q}: \ol{\gtc{G_i} \locProjD^{\pri+4} q})}_{q \in \tilde{q} \setminus \{r\}}
                            }
                            \bussUn[\ruleLabel{typ-recv}]{
                                \forall i \in I.~ \vdash 4_i \typInf \begin{array}[t]{@{}l@{}}
                                    \ci{\mu}{s}: \ol{\msgProp{S_i}} \parr^{\pri+1} \ol{(\gtc{G_i} \locProjD^{\pri+4} s)}, \ci{\mu}{r}: \msgProp{S_i} \tensor^{\pri+3} \ol{(\gtc{G_i} \locProjD^{\pri+4} r)},
                                    \\
                                    {(\ci{\mu}{q}: \ol{\gtc{G_i} \locProjD^{\pri+4} q})}_{q \in \tilde{q} \setminus \{s,r\}}
                                \end{array}
                            }
                            \bussUn[$\ruleLabel{typ-sel$\star$}^\ast$]{
                                \forall i \in I.~ \vdash 3_i \typInf \begin{array}[t]{@{}l@{}}
                                    \ci{\mu}{s}: \ol{\msgProp{S_i}} \parr^{\pri+1} \ol{(\gtc{G_i} \locProjD^{\pri+4} s)},
                                    {(\ci{\mu}{q}: \oplus^{\pri+2} \{i: \ol{(\gtc{G_i} \locProjD^{\pri+4} q)} \}_{i \in I})}_{q \in \deps_s},
                                    \\
                                    \ci{\mu}{r}: \msgProp{S_i} \tensor^{\pri+3} \ol{(\gtc{G_i} \locProjD^{\pri+4} r)},
                                    {(\ci{\mu}{q}: \oplus^{\pri+3} \{i: \ol{(\gtc{G_i} \locProjD^{\pri+4} q)} \}_{i \in I})}_{q \in \deps_r},
                                    \\
                                    {(\ci{\mu}{q}: \ol{\gtc{G_i} \locProjD^{\pri+4} q})}_{q \in \tilde{q} \setminus \deps_s \setminus \deps_r \setminus \{s,r\}}
                                \end{array}
                            }
                            \bussUn[\ruleLabel{typ-sel$\star$}]{
                                \forall i \in I.~ \vdash 2_i \typInf \begin{array}[t]{@{}l@{}}
                                    \ci{\mu}{s}: \ol{\msgProp{S_i}} \parr^{\pri+1} \ol{(\gtc{G_i} \locProjD^{\pri+4} s)},
                                    {(\ci{\mu}{q}: \oplus^{\pri+2} \{i: \ol{(\gtc{G_i} \locProjD^{\pri+4} q)} \}_{i \in I})}_{q \in \deps_s},
                                    \\
                                    \ci{\mu}{r}: \oplus^{\pri+2} \{i: \msgProp{S_i} \tensor^{\pri+3} \ol{(\gtc{G_i} \locProjD^{\pri+4} r)} \}_{i \in I},
                                    \\
                                    {(\ci{\mu}{q}: \oplus^{\pri+3} \{i: \ol{(\gtc{G_i} \locProjD^{\pri+4} q)} \}_{i \in I})}_{q \in \deps_r},
                                    \\
                                    {(\ci{\mu}{q}: \ol{\gtc{G_i} \locProjD^{\pri+4} q})}_{q \in \tilde{q} \setminus \deps_s \setminus \deps_r \setminus \{s,r\}}
                                \end{array}
                            }
                            \bussUn[\ruleLabel{typ-recv}]{
                                \vdash \mdm{\tilde{q}}{\gtc{G_s}} = 1 \typInf \begin{array}[t]{@{}lr@{}}
                                    \ci{\mu}{s}: \&^\pri \{i: \ol{\msgProp{S_i}} \parr^{\pri+1} \ol{(\gtc{G_i} \locProjD^{\pri+4} s)} \}_{i \in I},
                                    & \text{(cf.\ \eqref{eq:mpstAPCP:mSType})}
                                    \\
                                    \ci{\mu}{r}: \oplus^{\pri+2} \{i: \msgProp{S_i} \tensor^{\pri+3} \ol{(\gtc{G_i} \locProjD^{\pri+4} r)} \}_{i \in I},
                                    & \text{(cf.\ \eqref{eq:mpstAPCP:mRType})}
                                    \\
                                    {(\ci{\mu}{q}: \oplus^{\pri+2} \{i: \ol{(\gtc{G_i} \locProjD^{\pri+4} q)} \}_{i \in I})}_{q \in \deps_s}
                                    & \text{(cf.\ \eqref{eq:mpstAPCP:mDepSType})}
                                    \\
                                    {(\ci{\mu}{q}: \oplus^{\pri+3} \{i: \ol{(\gtc{G_i} \locProjD^{\pri+4} q)} \}_{i \in I})}_{q \in \deps_r}
                                    & \text{(cf.\ \eqref{eq:mpstAPCP:mDepRType})}
                                    \\
                                    {(\ci{\mu}{q}: \ol{\gtc{G_{i'}} \locProjD^{\pri+4} q})}_{q \in \tilde{q} \setminus \deps_s \setminus \deps_r \setminus \{s,r\}}
                                    & \text{(cf.\ \eqref{eq:mpstAPCP:mDepOtherType})}
                                \end{array}
                            }
                        \end{bussproof}
                    \]
                \ifthesis\else \end{mdframed} \fi

                \caption{Typing derivation used in the proof of \Cref{t:mpstAPCP:mediumTypes}.}
                \label{f:mpstAPCP:tderiii}
            \end{figure}

        \item
            \emph{Recursive definition}: $\gtc{G_s} = \mu Z \gtc. \gtc{G'}$ (\cref{li:mpstAPCP:mdmRecDef}).
            Let
            \begin{align}
                \tilde{q}' \deq \{q \in \tilde{q} \mid \gtc{G_s} \locProjD^\pri q \neq \bullet\}
                \label{eq:mpstAPCP:mRecActive}
            \end{align}
            (as on \cref{li:mpstAPCP:mdmRecComp}).
            The analysis depends on whether $\tilde{q}' = \emptyset$ or not.
            \begin{itemize}
                \item
                    If $\tilde{q}' = \emptyset$ (\cref{li:mpstAPCP:mdmRecInact}), let us take stock of the types expected for each of the orchestrator's channels.
                    For now, we omit the substitutions in the types.
                    \begin{align}
                        & \text{For each $q \in \tilde{q}$, for $\ci{\mu}{q}$ we expect}
                        &
                        \ol{\gtc{G_s} \locProjD^{\pri_{\globCtx{C}}} q}
                        &= \bullet.
                        \label{eq:mpstAPCP:mRecEmptyType}
                    \end{align}
                    Because all expected types are $\bullet$, the substitutions do not affect the types, so we can omit them altogether.

                    First we apply Rule~\ruleLabel{typ-inact}, giving us an arbitrary recursive context, thus the recursive context we need.
                    Then, we apply Rule~\ruleLabel{typ-end} for $\ci{\mu}{q}$ for each $q \in \tilde{q}$ (cf.\ \eqref{eq:mpstAPCP:mRecEmptyType}), and obtain the typing of $\mdm{\tilde{q}}{\gtc{G_s}}$ (omitting the recursive context):
                    \[
                        \vdash \mdm{\tilde{q}}{\gtc{G_s}} = \0 \typInf {(\ci{\mu}{q}: \bullet)}_{q \in \tilde{q}}.
                    \]

                \item
                    If $\tilde{q}' \neq \emptyset$ (\cref{li:mpstAPCP:mdmRecRet}), let us take stock of the types expected for each of the orchestrator's channels.
                    Note that, because of the recursive definition on $Z$ in $\gtc{G_s}$, there cannot be another recursive definition in the context $\globCtx{C}$ capturing the recursion variable~$Z$.
                    Therefore, by \Cref{d:mpstAPCP:binders}, $Z \notin \widetilde{X_{\globCtx{C}}}$.
                    \begin{align}
                        & \text{For each $q \in \tilde{q}'$,}
                        &
                        &
                        \nonumber \\
                        & \text{for $\ci{\mu}{q}$ we expect}
                        &
                        & \deepUnfold(\ol{\gtc{G_s} \locProjD^{\pri_{\globCtx{C}}} q},\ldots)
                        \nonumber \\
                        &
                        &
                        &= \deepUnfold(\ol{\mu Z . (\gtc{G'} \locProjD^{\pri_{\globCtx{C}}} q)},\ldots)
                        \nonumber \\
                        &
                        &
                        &= \deepUnfold(\mu Z . \ol{\gtc{G'} \locProjD^{\pri_{\globCtx{C}}} q},\ldots)
                        \nonumber \\
                        &
                        &
                        &= \mu Z . \deepUnfold(\ol{\gtc{G'} \locProjD^{\pri_C} q}, {(X,t_X,\ol{\gtc{G}_X \locProjD^{\pri_X} q})}_{X \in \widetilde{X_{\globCtx{C}}}}).
                        \label{eq:mpstAPCP:mRecRecType}
                        \\
                        & \text{For each $q \in \tilde{q} \setminus \tilde{q}'$,}
                        \nonumber \\
                        & \text{for $\ci{\mu}{q}$ we expect}
                        &
                        & \deepUnfold(\ol{\gtc{G_s} \locProjD^{\pri_{\globCtx{C}}} q},\ldots)
                        \nonumber \\
                        &
                        &
                        &= \deepUnfold(\bullet,\ldots)
                        = \bullet.
                        \label{eq:mpstAPCP:mRecOtherType}
                    \end{align}
                    We also need an assignment in the recursive context for every $X \in \widetilde{X_{\globCtx{C}}}$, but not for~$Z$.

                    Let $\globCtx{C'} = \globCtx{C}[\mu Z \gtc. \globHole]$.
                    Clearly, $\gtc{G'} \gSub{\globCtx{C'}} \gtc{G}$.
                    Let us establish some facts about the recursion binders, priorities, and active participants related to $\globCtx{C'}$, $\gtc{G'}$, and $Z$:
                    \begin{itemize}
                        \item
                            $\widetilde{X_{\globCtx{C'}}} = \ctxBind(\globCtx{C'}) = (\ctxBind(\globCtx{C}), Z) = (\widetilde{X_{\globCtx{C}}}, Z)$ (cf.\ \Cref{d:mpstAPCP:binders}).
                        \item
                            $\gtc{G}_Z = \recDef(Z,\gtc{G}) = \gtc{G'}$, as proven by the context $\globCtx{C'}$ (cf.\ \Cref{d:mpstAPCP:recDef}).
                        \item
                            $\widetilde{Y_Z} = \subBind(\mu Z \gtc. \gtc{G}_Z,\gtc{G}) = \ctxBind(\globCtx{C}) = \widetilde{X_{\globCtx{C}}}$.
                        \item
                            $\pri_{\globCtx{C'}} = \ctxPri(\globCtx{C'}) = \ctxPri(\globCtx{C}) = \pri_{\globCtx{C}}$, and $\pri_Z = \varPri(Z,\gtc{G}) = \ctxPri(\globCtx{C}) = \pri_{\globCtx{C}}$, and hence $\pri_{\globCtx{C'}} = \pri_Z$ (cf.\ \Cref{d:mpstAPCP:ctxPri}).
                        \item
                            $\tilde{q}_Z = \tilde{q}'$ (cf.\ \Cref{d:mpstAPCP:activeParts} and~\eqref{eq:mpstAPCP:mRecActive}).
                    \end{itemize}

                    Because $\widetilde{X_{\globCtx{C'}}} = (\widetilde{X_{\globCtx{C}}}, Z)$ and $\tilde{q}' = \tilde{q}_Z$, $\tilde{q}'$ is appropriate for the IH.
                    We apply the IH on $\globCtx{C'}$, $\gtc{G'}$, and $\tilde{q}'$ to obtain a typing for $\mdm{\tilde{q}'}{\gtc{G'}}$, where we immediately make use of the facts established above.
                    We given the assignment to $Z$ in the recursive context separate from those for the recursion variables in $\widetilde{X_{\globCtx{C}}}$.
                    Also, by \Cref{p:mpstAPCP:deepUnfold}, we can write the final unfolding on $Z$ in the types separately.
                    \begin{align*}
                        &{\left( X: {\big( \deepUnfold(\ol{\gtc{G_X} \locProjD^{\pri_X} q}, {(Y,t_Y,\ol{\gtc{G_Y} \locProjD^{\pri_Y} q})}_{Y \in \widetilde{Y_X}}) \big)}_{q \in \tilde{q}_X} \right)}_{X \in \widetilde{X_{\globCtx{C}}}},
                        \\
                        &Z: {\big( \deepUnfold(\ol{\gtc{G'} \locProjD^{\pri_{\globCtx{C}}} q}, {(X,t_X,\ol{\gtc{G}_X \locProjD^{\pri_X} q})}_{X \in \widetilde{X_{\globCtx{C}}}}) \big)}_{q \in \tilde{q}'}
                        \\
                        &\vdash \mdm{\tilde{q}'}{\gtc{G'}}
                        \\
                        &\typInf
                        {\left( \ci{\mu}{q}: \unfold^{t_Z}(\mu Z . \deepUnfold(\ol{\gtc{G'} \locProjD^{\pri_{\globCtx{C}}} q}, {(X,t_X,\ol{\gtc{G}_X \locProjD^{\pri_X} q})}_{X \in \widetilde{X_{\globCtx{C}}}})) \right)}_{q \in \tilde{q}'}
                    \end{align*}

                    By assumption, we have
                    \[
                        t_Z = \maxpr {\left( \deepUnfold(\ol{\gtc{G'} \locProjD^{\pri_{\globCtx{C}}} q}, {(X,t_X,\ol{\gtc{G}_X \locProjD^{\pri_X} q})}_{X \in \widetilde{X_{\globCtx{C}}}}) \right)}_{q \in \tilde{q}'} + 1,
                    \]
                    so $t_Z$ is clearly bigger than the maximum priority appearing in the types before unfolding.
                    Hence, we can apply Rule~\ruleLabel{typ-rec} to eliminate $Z$ from the recursive context, and to fold the types, giving the typing of $\mdm{\tilde{q}}{\gtc{G_s}} = \pRec Z({(\ci{\mu}{q})}_{q \in \tilde{q}'}) ; \mdm{\tilde{q}'}{\gtc{G'}}$:
                    \begin{align*}
                        &{\left( X: {\big( \deepUnfold(\ol{\gtc{G}_X \locProjD^{\pri_X} q}, {(Y,t_Y,\ol{\gtc{G}_Y \locProjD^{\pri_Y} q})}_{Y \in \widetilde{Y_X}}) \big)}_{q \in \tilde{q}_X} \right)}_{X \in \widetilde{X_{\globCtx{C}}}}
                        \\
                        &\vdash \mdm{\tilde{q}}{\gtc{G_s}} \typInf
                        {\left( \ci{\mu}{q}: \mu Z . \deepUnfold(\ol{\gtc{G'} \locProjD^{\pri_{\globCtx{C}}} q}, {(X,t_X,\ol{\gtc{G}_X \locProjD^{\pri_X} q})}_{X \in \widetilde{X_{\globCtx{C}}}}) \right)}_{q \in \tilde{q}'}
                    \end{align*}

                    In this typing, the type for $\ci{\mu}{q}$ for every $q \in \tilde{q}'$ concurs with~\eqref{eq:mpstAPCP:mRecRecType}.
                    For every $q \in \tilde{q} \setminus \tilde{q}'$, we can add the type for $\ci{\mu}{q}$ in~\eqref{eq:mpstAPCP:mRecOtherType} by applying Rule~\ruleLabel{typ-end}.
                    This proves the thesis.
            \end{itemize}

        \item
            \emph{Recursive call}: $\gtc{G_s} = Z$ (\cref{li:mpstAPCP:mdmRecCall}).

            Following similar reasoning as in the case of recursive call in the proof of \Cref{t:mpstAPCP:routerTypesGen}, let us take stock of the types we expect for our orchestrator's channels.
            \begin{align}
                & \text{For each $q \in \tilde{q}$,}
                \nonumber \\
                & \text{for $\ci{\mu}{q}$ we expect}
                &
                & \deepUnfold(\ol{\gtc{G_s} \locProjD^{\pri_{\globCtx{C}}} q},\ldots)
                \nonumber \\
                &
                &
                &= \deepUnfold(\ol{Z \locProjD^{\pri_{\globCtx{C}}} q},\ldots)
                \nonumber \\
                &
                &
                &= \deepUnfold(\ol{Z},\ldots)
                \nonumber \\
                &
                &
                &= \deepUnfold(Z, {(X,t_X,\ol{\gtc{G}_X \locProjD^{\pri_X} q})}_{X \in (\tilde{X}_1,Z,\widetilde{Y_Z})})
                \nonumber \\
                &
                &
                &= \deepUnfold(Z, {(X,t_X,\ol{\gtc{G}_X \locProjD^{\pri_X} q})}_{X \in (Z,\widetilde{Y_Z})})
                \nonumber \\
                &
                &
                &= \mu Z . (\lift{t_Z} \deepUnfold(\ol{\gtc{G}_Z \locProjD^{\pri_Z} q}, {(X,t_X,\ol{\gtc{G}_X \locProjD^{\pri_X} q})}_{X \in \widetilde{Y_Z}}))
                \label{eq:mpstAPCP:mRecCallType}
            \end{align}
            Also, we need an assignment in the recursive context for every $X \in \widetilde{X_{\globCtx{C}}}$.
            By \Cref{l:mpstAPCP:binderParts}, $\tilde{q} = \tilde{q}_Z$.
            Hence, for $Z$, the assignment should be as follows:
            \begin{align}
                Z: {\left( \deepUnfold(\ol{\gtc{G}_Z \locProjD^{\pri_Z} q}, {(X,t_X,\ol{\gtc{G}_X \locProjD^{\pri_X} q)}_{X \in \widetilde{Y_Z}}}) \right)}_{q \in \tilde{q}}
                \label{eq:mpstAPCP:mRecCallRec}
            \end{align}

            We apply Rule~\ruleLabel{typ-var} to obtain the typing of $\mdm{\tilde{q}}{\gtc{G_s}}$, where we make us the rule's allowance for an arbitrary recursive context up to the assignment to $Z$.
            Rule~\ruleLabel{typ-var} is applicable, because the types are recursive definitions on $Z$, concurring with the types assigned to $Z$, and lifted by a common lifter~$t_Z$.
            \[
                \begin{bussproof}
                    \bussAx[\ruleLabel{typ-var}]{
                        \begin{array}{@{}l@{}}
                            {\left( X: {\big( \deepUnfold(\ol{\gtc{G}_X \locProjD^{\pri_X} q}, {(Y,t_Y,\ol{\gtc{G}_Y \locProjD^{\pri_Y} q})}_{Y \in \widetilde{Y_X}}) \big)}_{q \in \tilde{q}_X} \right)}_{X \in \widetilde{X_{\globCtx{C}}} \setminus (Z)},
                            \\
                            Z: {\left( \deepUnfold(\ol{\gtc{G}_Z \locProjD^{\pri_Z} q}, {(X,t_X,\ol{\gtc{G}_X \locProjD^{\pri_X} q)}_{X \in \widetilde{Y_Z}}}) \right)}_{q \in \tilde{q}}
                            \\
                            \vdash \mdm{\tilde{q}}{\gtc{G_s}} = \pCall X<{(\ci{\mu}{q})}_{q \in \tilde{q}}> \typInf
                            \\
                            {\big( \ci{\mu}{q}: \mu Z . (\lift{t_Z} \deepUnfold(\ol{\gtc{G}_Z \locProjD^{\pri_Z} q}, {(X,t_X,\ol{\gtc{G}_X \locProjD^{\pri_X} q})}_{X \in \widetilde{Y_Z}})) \big)}_{q \in \tilde{q}}
                        \end{array}
                    }
                \end{bussproof}
            \]
            In this typing, the types of $\ci{\mu}{q}$ for each $q \in \tilde{q}$ concur with the expected types in~\eqref{eq:mpstAPCP:mRecCallType}, and the assignment to $Z$ in the recursive context concurs with~\eqref{eq:mpstAPCP:mRecCallRec}.
            This proves the thesis.
            \qedhere
    \end{itemize}
\end{proof}
\fi

\subsubsection{Orchestrators and Centralized Compositions of Routers are \texorpdfstring{\\}{} Behaviorally Equivalent}
\label{s:mpstAPCP:analysis:mediums:bisim}

First, we formalize what we mean by a centralized composition of routers, called a \emph{hub of routers}.
A hub of routers is just a specific composition of routers, formalized as the centralized composition of the routers of all a global type's participants synthesized from the global type:

\begin{definition}[Hub of a Global Type]
    \label{d:mpstAPCP:hub}
    Given global type $\gtc{G}$, we define the \emph{hub of routers} of~$\gtc{G}$ as follows:
    \[
        \hub_{\gtc{G}} \deq \pRes{\crt{p}{q} \crt{q}{p}}_{p,q \in \prt(\gtc{G})} \prod_{p \in \prt(\gtc{G})} \rtr{p}
    \]
\end{definition}

Hubs of routers can be typed using local projection (cf.\ \Cref{d:mpstAPCP:locProj}), identical to the typing of orchestrators (cf.\ \Cref{t:mpstAPCP:mediumTypes}):

\begin{theorem}
    \label{t:mpstAPCP:hubTypes}
    For relative well-formed global type $\gtc{G}$ and priority $\pri$,
    \[
        \emptyset \vdash \hub_{\globCtx{G}} \typInf {(\ci{\mu}{p}: \ol{(\gtc{G} \locProjD^\pri p)})}_{p \in \prt(\gtc{G})}.
    \]
\end{theorem}

\begin{proof}
    By the typability of routers (\Cref{t:mpstAPCP:routerTypes}) and the duality of the types of the endpoints connecting pairs of routers (\Cref{t:mpstAPCP:relPropDual}).
\end{proof}

\begin{figure}[t]
    \ifthesis\else \begin{mdframed} \fi
        \begin{mathpar}
            \begin{bussproof}[lts-send]
                \bussAx{
                    \pOut x[a,b] \xredd{\pOut x[a,b]} \0
                }
            \end{bussproof}
            \and
            \begin{bussproof}[lts-send-open]
                \bussAssume{
                    P \xredd{\pOut x[a,b]} P'
                }
                \bussUn{
                    \pRes{ya} \pRes{zb} P \xredd{\pRes{ya} \pRes{zb} \pOut x[a,b]} P'
                }
            \end{bussproof}
            \and
            \begin{bussproof}[lts-recv]
                \bussAx{
                    \pIn u(v,w) ; P \xredd{\pIn u(v,w)} P
                }
            \end{bussproof}
            \and
            \begin{bussproof}[lts-send-close]
                \bussAssume{
                    P \xredd{\pRes{ya} \pRes{zb} \pOut x[a,b]} P'
                }
                \bussAssume{
                    Q \xredd{\pIn u(v,w)} Q'
                }
                \bussBin{
                    P \| Q \xredd{x\|u} \pRes{yv} \pRes{zw} ( P' \| Q' )
                }
            \end{bussproof}
            \and
            \begin{bussproof}[lts-sel]
                \bussAx{
                    \pSel x[b] < j \xredd{\pSel x[b] < j} \0
                }
            \end{bussproof}
            \and
            \begin{bussproof}[lts-sel-open]
                \bussAssume{
                    P \xredd{\pSel x[b] < j} P'
                }
                \bussUn{
                    \pRes{zb} P \xredd{\pRes{zb} \pSel x[b] < j} P'
                }
            \end{bussproof}
            \and
            \begin{bussproof}[lts-bra]
                \bussAssume{
                    j \in I
                }
                \bussUn{
                    \pBra u(w) > \{ i : P_i \}_{i \in I} \xredd{\pBra u(w) > j} P_j
                }
            \end{bussproof}
            \and
            \begin{bussproof}[lts-sel-close]
                \bussAssume{
                    P \xredd{\pRes{zb} \pSel x[b] < j} P'
                }
                \bussAssume{
                    Q \xredd{\pBra u(w) > j} Q'
                }
                \bussBin{
                    P \| Q \xredd{x\|u} \pRes{zw} ( P' \| Q' )
                }
            \end{bussproof}
            \and
            \begin{bussproof}[lts-par-L]
                \bussAssume{
                    P \xredd{\alpha} Q
                }
                \bussAssume{
                    \bn(\alpha) \cap \fn(R) = \emptyset
                }
                \bussBin{
                    P \| R \xredd{\alpha} Q \| R
                }
            \end{bussproof}
            \and
            \begin{bussproof}[lts-par-R]
                \bussAssume{
                    P \xredd{\alpha} Q
                }
                \bussAssume{
                    \bn(\alpha) \cap \fn(R) = \emptyset
                }
                \bussBin{
                    R \| P \xredd{\alpha} R \| Q
                }
            \end{bussproof}
            \and
            \begin{bussproof}[lts-fwd]
                \bussAx{
                    \pRes{yz} ( \pFwd [x<>y] \| P ) \xredd{\tau} P \{ x/z \}
                }
            \end{bussproof}
            \and
            \begin{bussproof}[lts-comm]
                \bussAssume{
                    P \xredd{x\|u} Q
                }
                \bussUn{
                    \pRes{xu} P \xredd{\tau} Q
                }
            \end{bussproof}
            \and
            \begin{bussproof}[lts-res]
                \bussAssume{
                    P \xredd{\alpha} Q
                }
                \bussAssume{
                    y,y' \notin \fn(\alpha)
                }
                \bussBin{
                    \pRes{yy'} P \xredd{\alpha} \pRes{yy'} Q
                }
            \end{bussproof}
        \end{mathpar}
    \ifthesis\else \end{mdframed} \fi

    \caption{Labeled transition system for \APCP (cf.\ \Cref{d:mpstAPCP:lts}).}
    \label{f:mpstAPCP:lts}
\end{figure}

To state the behavioral equivalence of orchestrators and hubs of routers, we first define the specific behavioral equivalence we desire.
To this end, we first define a labeled transition system (\LTS) for \APCP:

\begin{definition}[\LTS for \APCP]
    \label{d:mpstAPCP:lts}
    We define the labels $\alpha$ for transitions for processes as follows:
    \begin{align*}
        \alpha &
        \begin{array}[t]{@{}l@{}lr@{\kern2em}l@{\kern1em}lr@{}}
            {} ::= {} &
            x\|y & \text{communication}
            & \sepr &
            \tau & \text{hidden}
            \\ \sepr* &
            \pOut x[a,b] & \text{send}
            & \sepr &
            \pRes{ya} \pRes{zb} \pOut x[a,b] & \text{bound send}
            \\ \sepr* &
            \pSel x[b] < j & \text{selection}
            & \sepr &
            \pRes{zb} \pSel x[b] < j & \text{bound selection}
            \\ \sepr* &
            \pIn u(v,w) & \text{receive}
            & \sepr &
            \pBra y(w) > j & \text{branch}
        \end{array}
    \end{align*}
    The relation \emph{labeled transition} ($P \xredd{\alpha} Q$) is then defined by the rules in \Cref{f:mpstAPCP:lts}.
\end{definition}

\begin{proposition}
    \label{p:mpstAPCP:redIsTau}
    $P \redd Q$ if and only if $P \xredd{\tau} Q$.
\end{proposition}

\noindent
As customary, we write $\Redd$ for the reflexive, transitive closure of $ \xredd{\tau}$, and we write $\xRedd{\alpha}$ for $\Redd \xredd{\alpha} \Redd$ if $\alpha \neq \tau$ and for $\Redd$ otherwise.

We can now define the behavioral equivalence we desire:

\begin{definition}[Weak Bisimilarity]
    \label{d:mpstAPCP:weakBisim}
    A binary relation $\mathbb{B}$ on processes is a \emph{weak bisimulation} if, for every $(P,Q) \in \mathbb{B}$,
    \begin{itemize}
        \item
            $P \xredd{\alpha} P'$ implies that there is $Q'$ such that $Q \xRedd{\alpha} Q'$ and $(P',Q') \in \mbb{B}$, and
        \item
            $Q \xredd{\alpha} Q'$ implies that there is $P'$ such that $P \xRedd{\alpha} P'$ and $(P',Q') \in \mbb{B}$.
    \end{itemize}
    Two processes $P$ and $Q$ are \emph{weakly bisimilar} if there exists a weak bisimulation $\mbb{B}$ such that $(P,Q) \in \mbb{B}$.
\end{definition}

Our equivalence result shall relate an orchestrator and a hub on a single but arbitrary channel.
In order to isolate such a channel, we place the orchestrator and hub of routers in an evaluation context consisting of restrictions and parallel compositions with arbitrary processes, such that it connects all but one of the orchestrator's or hub's channels.
For example, given a global type $\gtc{G}$ and implementations $\emptyset \vdash P_q \typInf \ci{\mu}{q}: \gtc{G} \locProjD^0 q$ for every participant $q \in \prt(\gtc{G}) \setminus \{p\}$, we could use the following evaluation context:
\[
    \evalCtx{E} \deq \pRes{\ci{\mu}{q} \ci{q}{\mu}}_{q \in \prt(\gtc{G}) \setminus \{p\}} \big( \prod_{q \in \prt(\gtc{G}) \setminus \{p\}} P_q \| \evalHole \big)
\]
Replacing the hole in this evaluation context with the orchestrator or hub of routers of $\gtc{G}$ leaves one channel free: the channel $\ci{\mu}{p}$ for the implementation of $p$.
Now we can observe the behavior of these two processes on $\ci{\mu}{p}$.

\ifmpstAPCP\else
In what follows we write $\mcl{O}_{\gtc{G}}^{\tilde{q}}$ instead of $\mdm{\tilde{q}}{\gtc{G}}$.
When we appeal to router and orchestrator synthesis, we often omit the parameter $\tilde{q}$.
That is, we write $\lmed{p}{}{\gtc{G}}$ instead of $\lmed{p}{\tilde{q}}{\gtc{G}}$, and $\mcl{O}_{\gtc{G}}$ instead of $\mcl{O}_{\gtc{G}}^{\tilde{q}}$.
\fi

\begin{restatable}{theorem}{tMpstAPCPMediumBisim}
    \label{t:mpstAPCP:mediumBisim}
    Suppose given a relative well-formed global type $\gtc{G}$.
    Let $\mcl{H}_{\gtc{G}}$ be the hub of routers of $\gtc{G}$ (\Cref{d:mpstAPCP:hub}) and take the orchestrator $\mdm{\prt(\gtc{G})}{\gtc{G}}$ of $\gtc{G}$ (\Cref{d:mpstAPCP:medium}).
    Let $p \in \prt(\gtc{G})$, and let $\evalCtx{E}$ be an evaluation context such that $\emptyset \vdash \evalCtx{E}[\mcl{H}_{\gtc{G}}] \typInf \ci{\mu}{p}: \ol{(\gtc{G} \locProjD^\pri p)}$.
    Then, $\evalCtx{E}[\mcl{H}_{\gtc{G}}]$ and $\evalCtx[\big]{E}[\mdm{\prt(\gtc{G})}{\gtc{G}}]$ are weakly bisimilar (\Cref{d:mpstAPCP:weakBisim}).
\end{restatable}

\ifmpstAPCP
\begin{sketch}
    By coinductively constructing a candidate relation $\mbb{B}$ and showing that it is a weak bisimulation and that $(\evalCtx{E}[\mcl{H}_{\gtc{G}}],\evalCtx[\big]{E}[\mdm{\prt(\gtc{G})}{\gtc{G}}]) \in \mbb{B}$.
\end{sketch}
\else
\ifthesis 
\ifmpstAPCP
We list the auxiliary definitions an results necessary to prove \Cref{t:mpstAPCP:mediumBisim}.
\else
We first give an intuition for the proof of \Cref{t:mpstAPCP:mediumBisim} and its ingredients, after which we give the proof and detail the ingredients.
\fi
The proof is by coinduction, i.e., by exhibiting a weak bisimulation $\B$ that contains the pair $(\evalCtx{E}[\hub_{\gtc{G}}],\evalCtx{E}[\mcl{O}_{\gtc{G}}])$.
\ifmpstAPCP\else
To construct $\B$ and prove that it is a weak bisimulation we require the following:
\fi
\begin{itemize}
    \ifmpstAPCP
    \item
        \Cref{n:mpstAPCP:ltsAPCP} sets up some useful notation for the LTS of \APCP (\Cref{d:mpstAPCP:lts}).
    \fi

    \item
        We define a function that, given a global type $\gtc{G}$ and a \emph{starting relation} $\B_0$, computes a corresponding \emph{candidate relation}.
        This function, given in \Cref{d:mpstAPCP:candidateBisim}, is denoted $\B(\gtc{G},\B_0,p)$.

        The function $\B(\gtc{G},\B_0,p)$ constructs a relation between processes by following labeled transitions on $\ci{\mu}{p}$ that concur with the expected behavior of $p$'s router and the orchestrator depending on the shape of $\gtc{G}$.
        For example, for $\gtc{G} = s {!} p \gtBraces{ \msg i<S_i> \gtc. \gtc{G_i} }_{i \in I}$, for each $i \in I$, the function constructs $\B_1^i$ containing the processes reachable from $\B_0$ through a transition labeled $\pSel* \ci{\mu}{p} < i$ (selection of the label chosen by $s$), and $\B_2^i$ containing the processes reachable from $\B_0$ through a transition labeled $\pOut* \ci{\mu}{p}[y]$ (output of the endpoint sent by $s$); the resulting relation then consists of $\B_0$ and, for each $i \in I$, $\B_1^i$ and $\B(\gtc{G_i},\B_2^i,p)$ (i.e., the candidate relation for $\gtc{G_i}$ starting with $B_2^i$).
        Since we are interested in a \emph{weak} bisimulation, the $\tau$-transitions of one process do not need to be simulated by related processes.
        Hence, e.g., if $(P,Q) \in \B_0$ and $P \xredd{\tau} P'$ and $Q \xredd{\tau} Q'$, then $\{(P,Q), (P',Q), (P,Q'), (P',Q')\} \subseteq \B(\gtc{G},\B_0,p)$.
        This way, we only \emph{synchronize} related processes when they can both take the same labeled transition.

    \item
        We intend to show that, if $\gtc{G} \lredd{\beta_1} \ldots \lredd{\beta_k} \gtc{G'}$, the function $\B(\gtc{G'},\B_0,p)$ constructs a weak bisimulation.
        However, for this to hold, the starting relation $\B_0$ cannot be arbitrary: the pairs of processes in $\B_0$ have to be reachable from $\evalCtx{E}[\hub_{\gtc{G}}]$ and $\evalCtx{E}[\mcl{O}_{\gtc{G}}]$ through labeled transitions that concur with the transitions from $\gtc{G}$ to $\gtc{G'}$.
        Moreover, the processes must have ``passed through'' evaluation contexts containing the router for $p$ at $\gtc{G'}$ and the orchestrator at $\gtc{G'}$.
        \Cref{d:mpstAPCP:consStartingRel} gives a \emph{consistent starting relation}, parametric on $k$, that satisfies these requirements.
        Note that for constructing the relation $\B$, we only need the following definition for $k = 0$.
        However, in the proof that $\B$ is a weak bisimulation we need to generalize it to $k \geq 0$ to assure that the starting relation of coinductive steps is consistent.

    \item
        \Cref{l:mpstAPCP:startTrans} states that processes in such a consistent starting relation follow a pattern of specific labeled transitions, passing through a context containing the router of $p$ or the orchestrator.

    \item
        \Cref{l:mpstAPCP:Bbisim} states that the relation obtained from $\B(\gtc{G'},\B_0,p)$ is a weak bisimulation, given the consistent starting relation $\B_0 = \<\gtc{G} \lredd{\beta_1} \ldots \lredd{\beta_k} \gtc{G'}, (\evalCtx{E}[\hub_{\gtc{G}}],\evalCtx{E}[\mcl{O}_{\gtc{G}}]), p\>$.

    \item
        The proof of \Cref{t:mpstAPCP:mediumBisim} on \Cpageref{proof:mpstAPCP:mediumBisim} follows from these definitions and results.
\end{itemize}

\begin{notation}
    \label{n:mpstAPCP:ltsAPCP}
    We adopt the following notational conventions.
    \begin{itemize}
        \item
            We write $\Proc$ to denote the set of all typable \APCP processes.
        \item
            In the LTS for \APCP (\Cref{d:mpstAPCP:lts}), we simplify labels: we write an overlined variant for output and selection (e.g., for $\pRes{ab} \pSel \ci{\mu}{p}[a] < \ell$ we write $\pSel* \ci{\mu}{p} < \ell$), and omit continuation channels for input and branching (e.g., for $\pBra \ci{\mu}{p}(a) > \ell$ we write $\pBra* \ci{\mu}{p} > \ell$).
        \item
            Also, we write $P \xRedd{\alpha_1 \ldots \alpha_n} Q$ rather than $P\xRedd{\alpha_1} P_1 \xRedd{\alpha_2} P_2 \ldots \xRedd{\alpha_n} Q$.
        \item
            We write $\vec{\alpha}$ to denote a sequence of labels, e.g., if $\vec{\alpha} = \alpha_1 \ldots \alpha_n$ then ${\xRedd{\vec{\alpha}}} = {\xRedd{\alpha_1 \ldots \alpha_n}}$.
            If $\vec{\alpha} = \epsilon$ (empty sequence), then ${\xRedd{\vec{\alpha}}} = {\Rightarrow}$.
        \end{itemize}
\end{notation}

\begin{definition}[Candidate Relation]
    \label{d:mpstAPCP:candidateBisim}
    Let $\gtc{G}$ be a global type and let $p$ be a participant of $\gtc{G}$.
    Also, let $\B_0 \subseteq \Proc \times \Proc$ denote a relation on processes.
    We define a \emph{candidate relation} for a weak bisimulation of the hub and orchestrator of $\gtc{G}$ observed on $\ci{\mu}{p}$ starting at $\B_0$, by abuse of notation denoted $\B(\gtc{G}, \B_0, p)$.
    The definition is inductive on the structure of $\gtc{G}$:
    \begin{itemize}
        \item
            $\gtc{G} = \tEnd$.
            Then $\B(\gtc{G},\B_0,p) = \B_0$.

        \item
            $\gtc{G} = s {!} r \gtBraces{ \msg i<S_i> \gtc. \gtc{G_i} }_{i \in I}$.
            We distinguish four cases, depending on the involvement of $p$:
            \begin{itemize}
                \item
                    $p = s$.
                    For every $i \in I$, let
                    \begin{align*}
                        \B_1^i &\deq \{(P_1,Q_1) \mid \exists (P_0,Q_0) \in \B_0.~ P_0 \xredd{\pBra* \ci{\mu}{p} > i} \Redd P_1 \wedge Q_0 \xredd{\pBra* \ci{\mu}{p} > i} \Redd Q_1\}
                        \\
                        \B_2^i &\deq \{(P_2,Q_2) \mid \exists (P_1,Q_1) \in \B_1^i.~ P_1 \xredd{\pIn \ci{\mu}{p}(y)} \Redd P_2 \wedge Q_1 \xredd{\pIn \ci{\mu}{p}(y)} \Redd Q_2\}
                    \end{align*}
                    Then
                    \[
                        \B(\gtc{G},\B_0,p) \deq \B_0 \cup \bigcup_{i \in I} (\B_1^i \cup \B(\gtc{G_i},\B_2^i,p)).
                    \]

                \item
                    $p = r$.
                    For every $i \in I$, let
                    \begin{align*}
                        \B_1^i &\deq \{(P_1,Q_1) \mid \exists (P_0,Q_0) \in \B_0.~ P_0 \xredd{\pSel* \ci{\mu}{p} < i} \Redd P_1 \wedge Q_0 \xredd{\pSel* \ci{\mu}{p} < i} \Redd Q_1\}
                        \\
                        \B_2^i &\deq \{(P_2,Q_2) \mid \exists (P_1,Q_1) \in \B_1 , y .~ P_1 \xredd{\pOut* \ci{\mu}{p}[y]} \Redd P_2 \wedge Q_1 \xredd{\pOut* \ci{\mu}{p}[y]} \Redd Q_2\}
                    \end{align*}
                    Then
                    \[
                        \B(\gtc{G},\B_0,p) \deq \B_0 \cup \bigcup_{i \in I} (\B_1^i \cup \B(\gtc{G_i},\B_2^i,p)).
                    \]

                \item
                    $p \notin \{s,r\}$ and $\hdep(p,s,\gtc{G})$ or $\hdep(p,r,\gtc{G})$.
                    For every $i \in I$, let
                    \[
                        \B_1^i \deq \{(P_1,Q_1) \mid \exists (P_0,Q_0) \in \B_0 .~ P_0 \xredd{\pSel* \ci{\mu}{p} < i} \Redd P_1 \wedge Q_0 \xredd{\pSel* \ci{\mu}{p} < i} \Redd Q_1\}
                    \]
                    Then
                    \[
                        \B(\gtc{G},\B_0,p) \deq \B_0 \cup \bigcup_{i \in I} \B(\gtc{G_i},\B_1^i,p).
                    \]

                \item
                    $p \notin \{s,r\}$ and neither $\hdep(p,s,\gtc{G})$ nor $\hdep(p,r,\gtc{G})$.
                    Then
                    \[
                        \B(\gtc{G},\B_0,p) \deq \B(\gtc{G_j},\B_0,p)
                    \]
                    for any $j \in I$.
            \end{itemize}

        \item
            $\gtc{G} = \mu X \gtc. \gtc{G'}$.
            Then $\B(\gtc{G},\B_0,p) \deq \B(\gtc{G'} \{ \mu X \gtc. \gtc{G'}/X \}, \B_0, p)$.

        \item
            $\gtc{G} = \tSkip \gtc. \gtc{G'}$.
            Then $\B(\gtc{G},\B_0,p) \deq \B(\gtc{G'},\B_0,p)$.
    \end{itemize}
\end{definition}

\begin{definition}[Consistent Starting Relation]
    \label{d:mpstAPCP:consStartingRel}
    Let $\gtc{G} \lredd{\beta_1} \ldots \lredd{\beta_k} \gtc{G'}$ (with $k \geq 0$) be a sequence of labeled transitions from $\gtc{G}$ to $\gtc{G'}$ including the intermediate global types (cf.\ \Cref{d:mpstAPCP:globalLts}) and let $p$ be a participant of $\gtc{G}$.
    Also, let $(P,Q)$ be a pair of initial processes.
    We define the \emph{consistent starting relation} for observing the hub and orchestrator of $\gtc{G'}$ on $\ci{\mu}{p}$ starting at $(P,Q)$ after the transitions from $\gtc{G}$ to $\gtc{G'}$, denoted $\<\gtc{G} \lredd{\beta_1} \ldots \lredd{\beta_k} \gtc{G'}, (P,Q), p\>$.
    The definition is inductive on the number $k$ of transitions:
    \begin{itemize}
        \item
            $k = 0$.
            Then $\<\gtc{G}, (P,Q), p\> \deq \{(P',Q') \mid P \Redd P' \wedge Q \Redd Q'\}$.

        \item
            $k = k'+1$.
            Then
            \begin{align*}
                & \<\gtc{G} \lredd{\beta_1} \ldots \lredd{\beta_{k'}} \gtc{G_{k'}} \lredd{\beta_k} \gtc{G_k}, (P,Q), p\>
                \\
                &\deq \{
                    (P_k,Q_k) \mid \begin{array}[t]{@{}l@{}}
                        \exists (P_{k'},Q_{k'}) \in \<\gtc{G} \lredd{\beta_1} \ldots \lredd{\beta_{k'}} \gtc{G_{k'}}, (P,Q), p\>.
                        \\
                        \big((\exists \evalCtx{C}.~ P_{k'} \xRedd{\vec{\alpha}} \evalCtx{C}[\lmed{p}{}{\gtc{G_k}}] \Redd P_k) \wedge (\exists \evalCtx{D}.~ Q_{k'} \xRedd{\vec{\alpha}} \evalCtx{D}[\mcl{O}_{\gtc{G_k}}] \Redd Q_k)\big)
                    \},
                \end{array}
            \end{align*}
            where $\vec{\alpha}$ depends on $\beta_k = s {!} r (\msg j<S_j>)$ and $\gtc{G_{k'}}$ (in unfolded form if $\gtc{G_{k'}} = \mu X \gtc. \gtc{G'_{k'}}$):
            \begin{itemize}
                \item
                    If $p = s$, then $\vec{\alpha} = \pBra* \ci{\mu}{p} > j , \pIn \ci{\mu}{p}(y)$.
                \item
                    If $p = r$, then $\vec{\alpha} = \pSel* \ci{\mu}{p} < j , \pOut* \ci{\mu}{p}[y]$.
                \item
                    If $p \notin \{s,r\}$ and $\hdep(p,s,\gtc{G_k})$ or $\hdep(p,r,\gtc{G_k})$, then $\vec{\alpha} = \pSel* \ci{\mu}{p} < j$.
                \item
                    If $p \notin \{s,r\}$ and neither $\hdep(p,s,\gtc{G_k})$ nor $\hdep(p,r,\gtc{G_k})$, then $\vec{\alpha} = \epsilon$.
            \end{itemize}
    \end{itemize}
\end{definition}

\begin{lemma}
    \label{l:mpstAPCP:startTrans}
    Let $\gtc{G}$ be a relative well-formed global type such that $\gtc{G} \lredd{\beta_1} \ldots \lredd{\beta_k} \gtc{G'}$ for $k \geq 0$ and let $p$ be a participant of $\gtc{G}$.
    Also, let $\evalCtx{E}$ be an evaluation context such that $\fn(\evalCtx{E}) = \{\ci{\mu}{p}\}$.
    Then there exists $\vec{\alpha}$ such that, for every $(P,Q) \in \<\gtc{G} \lredd{\beta_1} \ldots \lredd{\beta_k} \gtc{G'}, (\evalCtx{E}[\hub_{\gtc{G}}],\evalCtx{E}[\mcl{O}_{\gtc{G}}]), p\>$,
    \begin{itemize}
        \item
            $\evalCtx{E}[\hub_{\gtc{G}}] \xRedd{\vec{\alpha}} \evalCtx[\big]{C}[\lmed{p}{}{\gtc{G'}}] \Redd P$ where $\evalCtx{C}$ is an evaluation context without an output or selection on $\ci{\mu}{p}$; and
        \item
            $\evalCtx{E}[\mcl{O}_{\gtc{G}}] \xRedd{\vec{\alpha}} \evalCtx[\big]{D}[\mcl{O}_{\gtc{G'}}] \Redd Q$ where $\evalCtx{D}$ is an evaluation context without an output or selection on $\ci{\mu}{p}$.
    \end{itemize}
\end{lemma}

\begin{proof}
    By induction on $k$.
    In the base case ($k=0$), we have $\gtc{G} = \gtc{G'}$, so $\evalCtx{E}[\hub_{\gtc{G}}] = \evalCtx{C}[\lmed{p}{}{\gtc{G'}}] \Redd P$ and $\evalCtx{E}[\mcl{O}_{\gtc{G}}] = \evalCtx{D}[\mcl{O}_{\gtc{G'}}] \Redd Q$.

    For the inductive case ($k=k'+1$), we detail the representative case where
    \[
        \gtc{G} \lredd{\beta_1} \ldots \lredd{\beta_{k'}} \gtc{G_{k'}} = p {!} s \gtBraces{ \msg i<S_i> \gtc. \gtc{G'_i} }_{i \in I} \lredd{p {!} s (\msg i'<S_{i'}>)} \gtc{G'}
    \]
    for some $i' \in I$.
    By the IH, for every $(P_{k'},Q_{k'}) \in \<\gtc{G} \lredd{\beta_1} \ldots \lredd{\beta_{k'}} \gtc{G_{k'}}, (\evalCtx{E}[\hub_{\gtc{G}}],\evalCtx{E}[\mcl{O}_{\gtc{G}}]),p\>$, there exists $\vec{\alpha}'$ such that $\evalCtx{E}[\hub_{\gtc{G}}] \xRedd{\vec{\alpha}'} \evalCtx{C'}[\lmed{}{p}{\gtc{G_{k'}}}] \Redd P_{k'}$ and $\evalCtx{E}[\mcl{O}_{\gtc{G}}] \xRedd{\vec{\alpha}'} \evalCtx{D'}[\mcl{O}_{\gtc{G_{k'}}}] \Redd Q_{k'}$ where $\evalCtx{C'}$ and $\evalCtx{D'}$ are without output or selection on $\ci{\mu}{p}$.
    Take any 
    \[
        (P,Q) \in \<\gtc{G} \lredd{\beta_1} \ldots \lredd{\beta_{k'}} \gtc{G_{k'}} \lredd{s {!} p (\msg i'<S_{i'}>)} \gtc{G'}, (\evalCtx{E}[\hub_{\gtc{G}}], \evalCtx{E}[\mcl{O}_{\gtc{G}}]), p\>.
    \]
    By definition, there exists
    \[
        (P_{k'},Q_{k'}) \in \<\gtc{G} \lredd{\beta_1} \ldots \lredd{\beta_{k'}} \gtc{G_{k'}}, (\evalCtx{E}[\hub_{\gtc{G}}],\evalCtx{E}[\mcl{O}_{\gtc{G}}]),p\>
    \]
    such that
    \[
        P_{k'} \xRedd{\pSel* \ci{\mu}{p} < i' , \pOut* \ci{\mu}{p}[y]} \evalCtx{C}[\lmed{}{p}{\gtc{G'}}] \Redd P \text{ and } Q_{k'} \xRedd{\pSel* \ci{\mu}{p} < i' , \pOut* \ci{\mu}{p}[y]} \evalCtx{D}[\mcl{O}_{\gtc{G'}}] \Redd Q
    \]
    where there are no outputs or selection on $\ci{\mu}{p}$ in $\evalCtx{C}$ and $\evalCtx{D}$.
    Let $\vec{\alpha} = \vec{\alpha}' , \pSel* \ci{\mu}{p} < i' , \pOut* \ci{\mu}{p}[y]$.
    Then ${\evalCtx{E}[\hub_{\gtc{G}}] \xRedd{\vec{\alpha}} \evalCtx{C}[\lmed{p}{}{\gtc{G'}}] \Redd P}$ and $\evalCtx{E}[\mcl{O}_{\gtc{G}}] \xRedd{\vec{\alpha}} \evalCtx{D}[\mcl{O}_{\gtc{G'}}] \Redd Q$.
\end{proof}

\begin{lemma}
    \label{l:mpstAPCP:Bbisim}
    Let $\gtc{G}$ be a relative well-formed global type such that $\gtc{G} \lredd{\beta_1} \ldots \lredd{\beta_k} \gtc{G'}$ (with $k \geq 0$) and let $p$ be a participant of $\gtc{G}$.
    Also, let $\evalCtx{E}$ be an evaluation context such that $\fn(\evalCtx{E}) = \{\ci{\mu}{p}\}$.
    Then the relation $\B(\gtc{G'},\B_0)$, with $\B_0 = \<\gtc{G} \lredd{\beta_1} \ldots \lredd{\beta_k} \gtc{G'}, (\evalCtx{E}[\hub_{\gtc{G}}],\evalCtx{E}[\mcl{O}_{\gtc{G}}]), p\>$, is a weak bisimulation (cf.\ \Cref{d:mpstAPCP:weakBisim}).
\end{lemma}

\begin{proof}
    By coinduction on the structure of $\gtc{G'}$; there are four cases (communication, recursion, $\tSkip$, and $\tEnd$).
    We only detail the interesting case of communication, which is the only case which involves transitions with labels other than $\tau$.
    There are four subcases depending on the involvement of $p$ in the communication ($p$ is sender, $p$ is recipient, $p$ depends on the communication, or $p$ does not depend on the communication).
    In each subcase, the proof follows the same pattern, so as a representative case, we detail when $p$ is the recipient of the communication, i.e., $\gtc{G'} = s {!} p \gtBraces{ \msg i<S_i> \gtc. \gtc{G'_i} }_{i \in I}$.
    Recall
    \begin{align*}
        \lmed{p}{}{\gtc{G'}} &= \pBra* \crt{p}{s} > \big\{ i : \pSel* \ci{\mu}{p} < i \cdot {(\pSel* \crt{p}{q} < i)}_{q \in \deps} \cdot \pIn \crt{p}{s}(v) ; \pOut* \ci{\mu}{p}[w] \cdot ( \pFwd [v<>w] \| \lmed{p}{}{\gtc{G'_i}} ) \big\}_{i \in I},
        \tag{\Cref{alg:mpstAPCP:router} \cref{li:mpstAPCP:rtrRecv}}
        \\
        \lmed{s}{}{\gtc{G'}} &= \pBra* \ci{\mu}{p} > \big\{ i : \pSel* \crt{s}{p} < i \cdot {(\pSel* \crt{s}{q} < i)}_{q \in \deps} \cdot \pIn \ci{\mu}{s}(v) ; \pOut* \crt{s}{p}[w] \cdot ( \pFwd [v<>w] \| \lmed{s}{}{\gtc{G'_i}} ) \big\}_{i \in I},
        \tag{\Cref{alg:mpstAPCP:router} \cref{li:mpstAPCP:rtrSend}}
        \\
        \mcl{O}_{\gtc{G'}} &= \pBra* \ci{\mu}{s} > \{ i : \pSel* \ci{\mu}{p} < i \cdot {(\pSel* \ci{\mu}{q} < i)}_{q \in \deps} \cdot \pIn \ci{\mu}{s}(v) ; \pOut* \ci{\mu}{p}[w] \cdot ( \pFwd [v<>w] \| \mcl{O}_{\gtc{G'_i}} ) \}_{i \in I}.
        \tag{\Cref{alg:mpstAPCP:medium} \cref{li:mpstAPCP:mdmComm}}
    \end{align*}
    Let $\B = \B(\gtc{G'},\B_0)$.
    We have $\B = \B_0 \cup \bigcup_{i \in I}(\B_1^i \cup \B(\gtc{G'_i},\B_2^i))$ with $\B_1^i$ and $\B_2^i$ as defined above.
    Take any $(P,Q) \in \B$; we distinguish cases depending on the subset of $\B$ to which $(P,Q)$ belongs:
    \begin{itemize}
        \item
            $(P,Q) \in \B_0$.
            By \Cref{l:mpstAPCP:startTrans}, we have $\evalCtx{E}[\hub_{\gtc{G}}] \xRedd{\vec{\alpha}} \evalCtx{C}[\lmed{p}{}{\gtc{G'}}] \Redd P$ and \mbox{$\evalCtx{E}[\mcl{O}_{\gtc{G}}] \xRedd{\vec{\alpha}} \evalCtx{D}[\mcl{O}_{\gtc{G'}}] \Redd Q$}, where $\evalCtx{C}$ and $\evalCtx{D}$ do not contain an output or selection on~$\ci{\mu}{p}$.

            Suppose $P \xredd{\alpha} P'$; we need to exhibit a matching weak transition from $Q$.
            By assumption, there are no outputs or selections on $\ci{\mu}{p}$ in $\evalCtx{C}$ and $\evalCtx{D}$.
            Since there are no outputs or selections on $\ci{\mu}{p}$ in $\evalCtx{C}$, by definition of $\lmed{p}{}{\gtc{G'}}$, we need only consider two cases for $\alpha$:
            \begin{itemize}
                \item
                    $\alpha = \tau$.
                    We have $Q \Redd Q$, so $Q \xRedd{\tau} Q$.
                    Since $\evalCtx{C}[\lmed{p}{}{\gtc{G'}}] \Redd P'$ and $\evalCtx{D}[\mcl{O}_{\gtc{G'}}] \Redd Q$, we have $(P',Q) \in \B_0 \subseteq \B$.

                \item
                    $\alpha = \pSel* \ci{\mu}{p} < j$ for some $j \in I$.
                    To enable this transition, which originates from $p$'s router, somewhere in the $\tau$-transitions between $\evalCtx{C}[\lmed{p}{}{\gtc{G'}}]$ and $P$ the label $j$ was received on $\crt{p}{s}$, sent by the router of $s$ on $\crt{s}{p}$.
                    For this to happen, the label $j$ was received on $\ci{\mu}{s}$, sent from the context on $\ci{s}{\mu}$.
                    Since $\hub_{\gtc{G}}$ and $\mcl{O}_{\gtc{G}}$ are embedded in the same context, the communication of $j$ between $\ci{s}{\mu}$ and $\ci{\mu}{s}$ can also take place after a number of $\tau$-transitions from $\evalCtx{D}[\mcl{O}_{\gtc{G'}}]$, after which the selection of $j$ on $\ci{\mu}{p}$ becomes enabled.
                    Hence, since there are no outputs or selection on $\ci{\mu}{p}$ in $\evalCtx{D}$, we have $Q \Redd Q_0 \xredd{\pSel* \ci{\mu}{p} < j} Q'$.
                    We have $\evalCtx{D}[\mcl{O}_{\gtc{G'}}] \Redd Q_0$, so $(P,Q_0) \in \B_0$.
                    Since $P \xredd{\pSel* \ci{\mu}{p} < j} \Redd P'$ and $Q_0 \xredd{\pSel* \ci{\mu}{p} < j} \Redd Q'$, we have $(P',Q') \in \B_1^j \subseteq \B'$.
            \end{itemize}

            Now suppose $Q \xredd{\alpha} Q'$; we need to exhibit a matching weak transition from $P$.
            Again, we need only consider two cases for $\alpha$:
            \begin{itemize}
                \item
                    $\alpha = \tau$.
                    Analogous to the similar case above.

                \item
                    $\alpha = \pSel* \ci{\mu}{p} < j$ for some $j \in I$.
                    To enable this transition, which originates from the orchestrator, somewhere in the $\tau$-transitions between $\evalCtx{D}[\mcl{O}_{\gtc{G'}}]$ and $Q$ the label $j$ was received on $\ci{\mu}{s}$, sent from the context on $\ci{s}{\mu}$.
                    Hence, this communication can also take place after a number of transitions from $\evalCtx{E}[\hub_{\gtc{G}}]$, where the label is received by the router of $s$.
                    After this, from $\evalCtx{C}[\lmed{p}{}{\gtc{G'}}]$, the router of $s$ forwards~$j$ to $p$'s router (communication between $\crt{s}{p}$ and $\crt{p}{s}$), enabling the selection of $j$ on~$\ci{\mu}{p}$ in $p$'s router.
                    Hence, since there are no outputs or selections in $\evalCtx{C}$, we have $P \Redd P_0 \xredd{\pSel* \ci{\mu}{p} < j} P'$.
                    We have $\evalCtx{C}[\lmed{p}{}{\gtc{G'}}] \Redd P_0$, so $(P_0,Q) \in \B_0$.
                    Since $P_0 \xredd{\pSel* \ci{\mu}{p} < j} \Redd P'$ and $Q \xredd{\pSel* \ci{\mu}{p} < j} \Redd Q'$, we have $(P',Q') \in \B_1^j \subseteq \B$.
            \end{itemize}

        \item
            $(P,Q) \in \B_1^j$ for some $j \in I$.
            We have $\evalCtx{E}[\hub_{\gtc{G}}] \xRedd{\vec{\alpha}} \evalCtx{C}[\lmed{p}{}{\gtc{G'}}] \Redd P_0 \xredd{\pSel* \ci{\mu}{p} < j} \Redd P$ and ${\evalCtx{E}[\mcl{O}_{\gtc{G}}] \xRedd{\vec{\alpha}} \evalCtx{D}[\mcl{O}_{\gtc{G'}}] \Redd Q_0 \xredd{\pSel* \ci{\mu}{p} < j} \Redd Q}$ where $(P_0,Q_0) \in \B_0$.
            Since we have already observed the selection of $j$ on $\ci{\mu}{p}$ from both the hub and the orchestrator, we know that the routers of $p$ and $s$ are in branch $j$, and similarly the orchestrator is in branch $j$.

            Suppose $P \xrightarrow{\alpha} P'$.
            To exhibit a matching weak transition from $Q$ we only need to consider two cases for $\alpha$:
            \begin{itemize}
                \item
                    $\alpha = \tau$.
                    We have $Q \xRedd{\tau} Q$, and $P_0 \xredd{\pSel* \ci{\mu}{p} < j} \Redd P'$ and $Q_0 \xredd{\pSel* \ci{\mu}{p} < j} \Redd Q$, so $(P',Q) \in \B_1^j \subseteq \B$.

                \item
                    $\alpha = \pOut* \ci{\mu}{p}[y]$ for some $y$.
                    The observed output of some $y$ on $\ci{\mu}{p}$ must originate from $p$'s router.
                    This output is only enabled after receiving some $v$ over $\crt{p}{s}$, which must be sent by the router of $s$ over $\crt{s}{p}$.
                    The output by the router of $s$ is only enabled after receiving some $v$ over $\ci{\mu}{s}$, sent by the context over $\ci{s}{\mu}$.
                    Since the hub and the orchestrator are embedded in the same context, the communication of $v$ from~$\ci{s}{\mu}$ to $\ci{\mu}{s}$ can also occur (or has already occurred) for the orchestrator.
                    After this, the output of $y$ over $\ci{\mu}{p}$ is enabled in the orchestrator, i.e., $Q \Redd Q_1 \xredd{\pOut* \ci{\mu}{p}[y]} Q'$.
                    We have $Q_0 \xredd{\pSel* \ci{\mu}{p} < j} \Redd Q_1$, so $(P,Q_1) \in \B_1^j$.
                    Since $P \xredd{\pOut* \ci{\mu}{p}[y]} \Redd P'$ and $Q_1 \xredd{\pOut* \ci{\mu}{p}[y]} \Redd Q'$, we have $(P',Q') \in \B_2^j$.
                    By definition, $\B_2^j \subseteq \B(\gtc{G'_j},B_2^j) \subseteq \B$, so $(P',Q') \in \B$.
            \end{itemize}

            Now suppose $Q \xredd{\alpha} Q'$.
            To exhibit a matching weak transition from $P$ we only need to consider two cases for $\alpha$:
            \begin{itemize}
                \item
                    $\alpha = \tau$.
                    Analogous to the similar case above.

                \item
                    $\alpha = \pOut* \ci{\mu}{p}[y]$ for some $y$.
                    The observed output of some $y$ on $\ci{\mu}{p}$ must originate from the orchestrator.
                    This output is only enabled after receiving some $v$ over~$\ci{\mu}{s}$, sent by the context of $\ci{s}{\mu}$.
                    Since the hub and the orchestrator are embeded in the same context, the communication of $v$ from $\ci{s}{\mu}$ to $\ci{\mu}{s}$ can also occur (or has already occurred) for the router of $s$.
                    After this, the router of $s$ sends another channel~$v'$ over $\crt{s}{p}$, received by $p$'s router on $\crt{p}{s}$.
                    This enables the output of $y$ on $\ci{\mu}{p}$ by $p$'s router, i.e., $P \Redd P_1 \xredd{\pOut* \ci{\mu}{p}[y]} P'$.
                    We have $P_0 \xredd{\pSel* \ci{\mu}{p} < j} \Redd P_1$, so $(P_1,Q) \in \B_1^j$.
                    Since $P \xredd{\pOut* \ci{\mu}{p}[y]} \Redd P'$ and $Q_1 \xredd{\pOut* \ci{\mu}{p}[y]} \Redd Q'$, we have $(P',Q') \in \B_2^j$.
                    As above, this implies that $(P',Q') \in \B$.
            \end{itemize}

        \item
            For some $j \in I$, $(P,Q) \in \B(\gtc{G'_j},\B_2^j)$.
            The thesis follows from proving that $\B(\gtc{G'_j},\B_2^j)$ is a weak bisimulation.
            For this, we want to appeal to the coinduction hypothesis, so we have to show that $\B_2^j = \<\gtc{G} \lredd{\beta_1} \ldots \lredd{\beta_k} G' \lredd{s {!} p (\msg j<S_j>)} \gtc{G'_j}, (\evalCtx{E}[\hub_{\gtc{G}}],\evalCtx{E}[\mcl{O}_{\gtc{G}}]), p\>$.
            We prove that $(P_2,Q_2) \in \B_2^j$ if and only if $(P_2,Q_2) \in \<\gtc{G} \lredd{\beta_1} \ldots \lredd{\beta_k} \gtc{G'} \lredd{s {!} p (\msg j<S_j>)} \gtc{G'_j}, (\evalCtx{E}[\hub_{\gtc{G}}],\evalCtx{E}[\mcl{O}_{\gtc{G}}]), p\>$, i.e., we prove both directions of the bi-implication:
            \begin{itemize}
                \item
                    Take any $(P_2,Q_2) \in \B_2^j$.
                    We have $\evalCtx{E}[\hub_{\gtc{G}}] \xRedd{\vec{\alpha}} \evalCtx{C}[\lmed{p}{}{\gtc{G'}}] \Redd P_0 \xredd{\pSel* \ci{\mu}{p} < j} \Redd P_1 \xredd{\pOut* \ci{\mu}{p}[y]} \Redd P_2$ and $\evalCtx{E}[\mcl{O}_{\gtc{G}}] \xRedd{\vec{\alpha}} \evalCtx{D}[\mcl{O}_{\gtc{G'}}] \Redd Q_0 \xredd{\pSel* \ci{\mu}{p} < j} \Redd Q_1 \xredd{\pOut* \ci{\mu}{p}[y]} \Redd Q_2$, where $(P_0,Q_0) \in \B_0$ and $(P_1,P_1) \in \B_1^j$.

                    By definition, somewhere during the transitions from $\evalCtx{C}[\lmed{p}{}{\gtc{G'}}]$ to $P_1$, we find $\evalCtx{C'}[\lmed{p}{}{\gtc{G'_j}}]$, which may then further reduce by $\tau$-transitions towards $P_2$.
                    As soon as we do find $\evalCtx{C'}[\lmed{p}{}{\gtc{G'_j}}]$, the output on $\ci{\mu}{p}$ is available, and the selection on $\ci{\mu}{p}$ has already occurred or is still available.
                    Because they are asynchronous actions, we can observe the selection and output on $\ci{\mu}{p}$ as soon as they are available, before further reducing $p$'s router.
                    Hence, we can observe
                    \[
                        \evalCtx{C}[\lmed{p}{}{\gtc{G'}}] \Redd \xredd{\pSel* \ci{\mu}{p} < j} \Redd \xredd{\pOut* \ci{\mu}{p}[y]} \Redd \evalCtx{C''}[\lmed{p}{}{\gtc{G'_j}}] \Redd P_2,
                    \]
                    i.e.,
                    \[
                        {\evalCtx{E}[\hub_{\gtc{G}}] \xRedd{\vec{\alpha}} \evalCtx{C}[\lmed{p}{}{\gtc{G'}}] \xRedd{{\pSel* \ci{\mu}{p} < j} , {\pOut* \ci{\mu}{p}[y]}} \evalCtx{C''}[\lmed{p}{}{\gtc{G'_j}}] \Redd P_2}.
                    \]
                    By definition, $\lmed{p}{}{\gtc{G'_j}}$ has no output or selection on $\ci{\mu}{p}$ available, so there are no outputs or selections on $\ci{\mu}{p}$ in $\evalCtx{C''}$.

                    By a similar argument, we can observe $\evalCtx{D}[\mcl{O}_{\gtc{G'}}] \Redd \xredd{\pSel* \ci{\mu}{p} < j} \Redd \xredd{\pOut* \ci{\mu}{p}[y]} \Redd \evalCtx{D''}[\mcl{O}_{\gtc{G'_j}}] \Redd Q_2$, i.e., $\evalCtx{E}[\mcl{O}_{\gtc{G}}] \xRedd{\vec{\alpha}} \evalCtx{D}[\mcl{O}_{\gtc{G'}}] \xRedd{{\pSel* \ci{\mu}{p} < j} , {\pOut* \ci{\mu}{p}[y]}} \evalCtx{D''}[\mcl{O}_{\gtc{G'_j}}] \Redd Q_2$.
                    Also in this case, there are no outputs or selections on $\ci{\mu}{p}$ in $\evalCtx{D''}$.

                    By assumption and definition,
                    \[
                        (\evalCtx{C''}[\lmed{p}{}{\gtc{G'}}],\evalCtx{D''}[\mcl{O}_{\gtc{G'}}]) \in \B_0 = \<\gtc{G} \lredd{\beta_1} \ldots \lredd{\beta_k} \gtc{G'}, (\evalCtx{E}[\hub_{\gtc{G}}],\evalCtx{E}[\mcl{O}_{\gtc{G}}]), p\>.
                    \]
                    Hence, by definition, $(P_2,Q_2) \in \<\gtc{G} \lredd{\beta_1} \ldots \lredd{\beta_k} \gtc{G'} \lredd{s {!} p (\msg j<S_j>)} \gtc{G'_j}, (\evalCtx{E}[\hub_{\gtc{G}}],\evalCtx{E}[\mcl{O}_{\gtc{G}}]), p\>$.

                \item
                    Take any $(P,Q) \in \<\gtc{G} \lredd{\beta_1} \ldots \lredd{\beta_k} \gtc{G'} \lredd{s {!} p (\msg j<S_j>)} \gtc{G'_j}, (\evalCtx{E}[\hub_{\gtc{G}}],\evalCtx{E}[\mcl{O}_{\gtc{G}}]), p\>$.
                    By definition, there are $(P',Q') \in \<\gtc{G} \lredd{\beta_1} \ldots \lredd{\beta_k} \gtc{G'}, (\evalCtx{E}[\hub_{\gtc{G}}],\evalCtx{E}[\mcl{O}_{\gtc{G}}]), p\>$ such that
                    \[
                        P' \xRedd{\pSel* \ci{\mu}{p} < j , \pOut* \ci{\mu}{p}[y]} \evalCtx{C}[\lmed{p}{}{\gtc{G'}}] \Redd P
                    \]
                    and $Q' \xRedd{\pSel* \ci{\mu}{p} < j , \pOut* \ci{\mu}{p}[y]} \evalCtx{D}[\mcl{O}_{\gtc{G'}}] \Redd Q$.
                    Since,
                    \[
                        \B_0 = \<\gtc{G} \lredd{\beta_1} \ldots \lredd{\beta_k} \gtc{G'}, (\evalCtx{E}[\hub_{\gtc{G}}],\evalCtx{E}[\mcl{O}_{\gtc{G}}]), p\>,
                    \]
                    by definition $(P,Q) \in \B_2^j$.
                    \qedhere
            \end{itemize}
    \end{itemize}
\end{proof}

\tMpstAPCPMediumBisim*

\begin{proof}
    \label{proof:mpstAPCP:mediumBisim}
    Let $\B = \B(\gtc{G},\B_0)$, where $\B_0 = \<\gtc{G}, (\evalCtx{E}[\hub_{\gtc{G}}],\evalCtx{E}[\mcl{O}_{\gtc{G}}]), p\>$.
    By \Cref{l:mpstAPCP:Bbisim}, $\B$ is a weak bisimulation.
    Because $(\evalCtx{E}[\hub_{\gtc{G}}],\evalCtx{E}[\mcl{O}_{\gtc{G}}]) \in \B_0 \subseteq \B$, it then follows that $\evalCtx{E}[\hub_{\gtc{G}}]$ and $\evalCtx{E}[\mcl{O}_{\gtc{G}}]$ are weakly bisimilar.
\end{proof}

\else  \fi
\fi

\section{Routers in Action}
\label{s:routersInAction}

We demonstrate our router-based analysis of global types by means of several examples.
First, in \Cref{s:mpstAPCP:examples:intrl,s:mpstAPCP:examples:deleg} we consider two simple protocols, illustrating the different components of our approach and our support for delegation and interleaving.
Then, in \Cref{s:mpstAPCP:examples:GAuth} we revisit~$\gtc{G_{\sff{auth}}}$ from \Cref{s:mpstAPCP:intro} to illustrate how our analysis supports also more complex protocols featuring also non-local choices and recursion.

\subsection{Delegation and Interleaving}
\label{s:mpstAPCP:examples:intrl}

We illustrate our analysis by considering a global type with delegation and interleaving, based on an example by Toninho and Yoshida~\cite[Ex.~6.9]{journal/toplas/ToninhoY18}.
Consider the global type:
\[
    \gtc{G_{\sff{intrl}}} \deq p {!} q (\msg 1<{!}\sff{int} . \tEnd>) \gtc. r {!} t (\msg 2<\sff{int}>) \gtc. p {!} q (3) \gtc. \tEnd
\]
Following Toninho and Yoshida~\cite{journal/toplas/ToninhoY18}, we define implementations of the roles of the four participants ($p,q,r,t$) of $\gtc{G_{\sff{intrl}}}$ using three processes ($P_1$, $P_2$, and $P_3$):
$P_2$ and $P_3$ implement the roles of $q$ and $r$, respectively, and $P_1$ interleaves the roles of $p$ and $t$ by sending a channel $s$ to $q$ and receiving an \sff{int} value $v$ from $r$, which it should forward to $q$ over $s$.
\begin{align*}
    \emptyset &\vdash P_1
    \deq \pSel* \ci{p}{\mu} < 1 \cdot \pOut* \ci{p}{\mu}[s] \cdot ( \pBra* \ci{t}{\mu} > \{ 2 : \pIn \ci{t}{\mu}(v) ; \pOut* s[w] \cdot \pFwd [v<>w] \} \| \pSel* \ci{p}{\mu} < 3 \cdot \pOut* \ci{p}{\mu}[z] \cdot \0 )
    \\
    &\typInf \ci{p}{\mu}: \oplus^0 \big\{1: \msgProp{{!}\sff{int} . \tEnd} \tensor^1 \oplus^8 \{3: \bullet \tensor^9 \bullet \} \big\},
    \ci{t}{\mu}: \&^6 \{2: \bullet \parr^7 \bullet \}
    \\[4pt]
    \emptyset &\vdash P_2
    \deq \pBra* \ci{q}{\mu} > \{ 1 : \pIn \ci{q}{\mu}(y) ; \pIn y(x) ; \pBra* \ci{q}{\mu} > \{ 3 : \pIn \ci{q}{\mu}(u) ; \0 \} \}
    \typInf \ci{q}{\mu}: \&^2 \big\{1: \ol{\msgProp{{!}\sff{int} . \tEnd}} \parr^3 \&^{10} \{3: \bullet \parr^{11} \bullet \} \big\}
    \\[4pt]
    \emptyset &\vdash P_3
    \deq \pSel* \ci{r}{\mu} < 2 \cdot \pOut* \ci{r}{\mu}[\bm{33}] \cdot \0
    \typInf \ci{r}{\mu}: \oplus^4 \{2: \bullet \tensor^5 \bullet \}
\end{align*}
where `$\bm{33}$' denotes a closed channel endpoint representing the number ``$33$''.

To prove that $P_1$, $P_2$, and $P_3$ correctly implement $\gtc{G_{\sff{intrl}}}$, we compose them with the routers synthesized from $\gtc{G_{\sff{intrl}}}$.
For example, the routers for $p$ and $t$, to which $P_1$ will connect, are as follows (omitting curly braces for branches on a single label):
\begin{align*}
    \rtr{p}
    &= \pBra* \ci{\mu}{p} > 1 ; \pSel* \crt{p}{q} < 1 \cdot \pIn \ci{\mu}{p}(s) ; \pOut* \crt{p}{q}[s'] \cdot ( \pFwd [s<>s'] \| \pBra* \ci{\mu}{p} > 3 ; \pSel* \crt{p}{q} < 3 \cdot \pIn \ci{\mu}{p}(z) ; \pOut* \crt{p}{q}[z'] \cdot ( \pFwd [z<>z'] \| \0 ) )
    \\
    \rtr{t}
    &= \pBra* \crt{t}{r} > 2 ; \pSel* \ci{\mu}{t} < 2 \cdot \pIn \crt{t}{r}(v) ;  \pOut* \ci{\mu}{t}[v'] \cdot ( \pFwd [v<>v'] \| \0 )
\end{align*}

We assign values to the priorities in $\msgProp{{!}\sff{int} . \tEnd} = \bullet \tensor^\pri \bullet$ to ensure that $P_1$ and $P_2$ are well-typed;
assigning $\pri = 8$ works, because the send on $s$ in $P_1$ occurs after the receive on $\ci{t}{\mu}$ (which has priority 6--7) and the receive on $y$ in $P_2$ occurs before the second branch on $\ci{q}{\mu}$ (which has priority 10--11).

The types assigned to $\ci{p}{\mu}$ and $\ci{t}{\mu}$ in $P_1$ coincide with $(\gtc{G_{\sff{intrl}}} \locProjD^0 p)$ and $(\gtc{G_{\sff{intrl}}} \locProjD^0 t)$, respectively (cf.\ \Cref{d:mpstAPCP:locProj}).
Therefore, by \Cref{t:mpstAPCP:routerTypes}, the process $P_1$ connects to the routers for~$p$ and $t$: $\pRes{\ci{p}{\mu} \ci{\mu}{p}} \pRes{\ci{t}{\mu} \ci{\mu t}{}} ( P_1 \| \rtr{p} \| \rtr{t} )$ is well-typed.
Similarly, $\pRes{\ci{q}{\mu} \ci{\mu}{q}} ( P_2 \| \rtr{q} )$ and \mbox{$\pRes{\ci{r}{\mu} \ci{\mu}{r}} ( P_3 \| \rtr{r} )$} are well-typed.

The composition of these routed implementations results in the following network:
\begin{align*}
    N_{\sff{intrl}} \deq \begin{array}{@{}c@{}}
        \pRes{\crt{p}{q} \crt{q}{p}} \pRes{\crt{p}{r} \crt{r}{p}}
        \\
        \pRes{\crt{p}{t} \crt{t}{p}} \pRes{\crt{q}{r} \crt{r}{q}}
        \\
        \pRes{\crt{q}{t} \crt{t}{q}} \pRes{\crt{r}{t} \crt{t}{r}}
    \end{array} \left(
        \begin{array}{@{}l@{}}
            \phantom{{} \| {}} \pRes{\ci{p}{\mu} \ci{\mu}{p}} \pRes{\ci{t}{\mu} \ci{\mu}{t}} (P_1 \| \rtr{p} \| \rtr{t})
            \\
            {} \| \pRes{\ci{q}{\mu} \ci{\mu}{q}} (P_2 \| \rtr{q})
            \\
            {} \| \pRes{\ci{r}{\mu} \ci{\mu}{r}} (P_3 \| \rtr{r})
        \end{array}
    \right)
\end{align*}
We have $N_{\sff{intrl}} \in \sys(\gtc{G_{\sff{intrl}}})$ (cf.\ \Cref{d:mpstAPCP:networks}), so, by \Cref{t:mpstAPCP:globalDf}, $N_{\sff{intrl}}$ is deadlock-free and, by \Cref{t:mpstAPCP:completeness,t:mpstAPCP:soundness}, it correctly implements $\gtc{G_{\sff{intrl}}}$.

\subsection{Another Example of Delegation}
\label{s:mpstAPCP:examples:deleg}

Here, we further demonstrate our support for interleaving, showing how a participant can delegate the rest of its interactions in a protocol.
The following global type formalizes a protocol in which a Client ($c$) asks an online Password Manager ($p$) to login with a Server~($s$):
\[
    \gtc{G_{\sff{deleg}}} \deq c {!} p (\msg \sff{login}<S>) \gtc. \gtc{G'_{\sff{deleg}}}
\]
where
\begin{align*}
    S
    &\deq {!}({?} \sff{bool} . \tEnd) . S'
    \\
    S'
    &\deq \& \{ \sff{passwd}: {?} \sff{str} . \oplus \{ \sff{auth}: {!} \sff{bool} . \tEnd \} \}
    \\
    \gtc{G'_{\sff{deleg}}}
    &\deq c {!} s (\msg \sff{passwd}<\sff{str}>) \gtc. s {!} c (\msg \sff{auth}<\sff{bool}>) \gtc. \tEnd
\end{align*}
Here $S'$ expresses the type of $\rtr{c}$'s channel endpoint $\ci{\mu}{c}$.
This means that we can give implementations for $c$ and $p$ such that $c$ can send its channel endpoint $\ci{c}{\mu}$ to $p$, after which $p$ logs in with $s$ in $c$'s place, forwarding the authorization boolean received from $s$ to $c$.
Giving such implementations is relatively straightforward, demonstrating the flexibility of our global types and analysis using \APCP and routers.

Using local projection, we can compute a type for $c$'s implementation to safely connect with its router
\[
    \gtc{G_{\sff{deleg}}} \locProjD^0 c = \oplus^{0} \{ \sff{login}: \msgProp{S} \tensor^{1} (\gtc{G'_{\sff{deleg}}} \locProjD^4 c) \}
\]
where
\begin{align*}
    \msgProp{S}
    &= (\bullet \parr^{\pri} \bullet) \tensor^{\kappa} \msgProp{S'}
    \\
    \msgProp{S'}
    &= \&^{\pi} \{ \sff{passwd}: \bullet \parr^{\rho} \oplus^{\delta} \{ \sff{auth}: \bullet \tensor^{\phi} \bullet \} \}
    \\
    \gtc{G'_{\sff{deleg}}} \locProjD^4 c
    &= \oplus^{4} \{ \sff{passwd}: \bullet \tensor^{5} \&^{10} \{ \sff{auth}: \bullet \parr^{11} \bullet \} \}
\end{align*}
Notice how $\ol{\msgProp{S'}} = \gtc{G'_{\sff{deleg}}} \locProjD^4 c$, given the assignments $\pi = 4, \rho = 5, \delta = 10, \phi = 11$.

We can use these types to guide the design of a process implementation for $c$.
Consider:
\[
    \emptyset \vdash Q \deq \pSel* \ci{c}{\mu} < \sff{login} \cdot \pOut* \ci{c}{\mu}[u] \cdot \pOut* u[v] \cdot ( \pFwd [u<>\ci{c}{\mu}] \| \pIn v(a) ; \0 ) \typInf \ci{c}{\mu}: \gtc{G_{\sff{deleg}}} \locProjD^0 c
\]
This implementation is interesting: after the first exchange in $\gtc{G_{\sff{deleg}}}$---sending a fresh channel~$u$ (to $p$)---$c$~sends another fresh channel $v$ over $u$; then, $c$ delegates the rest of its exchanges in $\gtc{G'_{\sff{deleg}}}$ by forwarding all traffic on $\ci{c}{\mu}$ over $u$; in the meantime, $c$ awaits an authorization boolean over $v$.

Again, using local projection, we can compute a type for $p$'s implementation to connect with its router:
\[
    \gtc{G_{\sff{deleg}}} \locProjD^0 p = \&^{2} \{ \sff{login}: \ol{\msgProp{S}} \parr^{3} \bullet \}
\]
We can then use it to type the following implementation for $p$:
\[
    \emptyset \vdash P \deq \pBra* \ci{p}{\mu} > \left\{
        \begin{array}{@{}l@{}l@{}}
            \sff{login}:{}
            & \pIn \ci{p}{\mu}(\ci{c}{\mu}) ; \pIn \ci{c}{\mu}(v) ;
            \\
            & \pSel* \ci{c}{\mu} < \sff{passwd} \cdot \pOut* \ci{c}{\mu}[\bm{pwd123}] \cdot
            \\
            & \pBra* \ci{c}{\mu} > \{ \sff{auth}: \pIn \ci{c}{\mu}(a) ; \pOut* v[a'] \cdot \pFwd [a<>a'] \}
        \end{array}
    \right\} \typInf \ci{p}{\mu}: \gtc{G_{\sff{deleg}}} \locProjD^0 p
\]
In this implementation, $p$ receives a channel $\ci{c}{\mu}$ (from $c$) over which it first receives a channel~$v$.
Then, it behaves over $\ci{c}{\mu}$ according to $c$'s role in $\gtc{G'_{\sff{deleg}}}$.
Finally, $p$ forwards the authorization boolean received from $s$ over $v$, effectively sending the boolean to $c$.

Given an implementation for $s$, say $\emptyset \vdash S \typInf \ci{s}{\mu}: \gtc{G_{\sff{deleg}}} \locProjD^0 s$, what remains is to assign values to the remaining priorities in $\msgProp{S}$: assigning $\pri = 12, \kappa = 4$ works.
Now, we can compose the implementations $P$, $Q$ and $S$ with their respective routers and then compose these routed implementations together to form a deadlock-free network of $\gtc{G_{\sff{deleg}}}$.
This way, e.g., the router for $c$ is as follows (again, omitting curly braces for branches on a single label):
\begin{align*}
    \rtr{c}
    &= \pBra* \ci{\mu}{c} > \sff{login} ; \pSel* \crt{c}{p} < \sff{login} \cdot \pIn \ci{\mu}{c}(u) ; \pOut* \crt{c}{p}[u'] \cdot (
    \\
    &\phantom{{}={}} \quad \pFwd [u<>u'] \| \pBra* \ci{\mu}{c} > \sff{passwd} ; \pSel* \crt{c}{s} < \sff{passwd} \cdot \pIn \ci{\mu}{c}(v) ; \pOut* \crt{c}{s}[v'] \cdot (
    \\
    &\phantom{{}={}} \qquad \pFwd [v<>v'] \| \pBra* \crt{c}{s} > \sff{auth} ; \pSel* \ci{\mu}{c} < \sff{auth} \cdot \pIn \crt{c}{s}(w) ; \pOut* \ci{\mu}{c}[w'] \cdot ( \pFwd [w<>w'] \| \0 ) ) )
\end{align*}
Interestingly, the router is agnostic of the fact that the endpoint $u$ it receives over $\ci{\mu}{c}$ is in fact the opposite endpoint of the channel formed by $\ci{\mu}{c}$.

\subsection{The Authorization Protocol in Action}
\label{s:mpstAPCP:examples:GAuth}

Let us repeat $\gtc{G_{\sff{auth}}}$ from \Cref{s:mpstAPCP:intro}:
\[
    \gtc{G_{\sff{auth}}} = \mu X \gtc. s {!} c \gtBraces*{
        \begin{array}{@{}l@{}}
            \sff{login} \gtc. c {!} a (\msg \sff{passwd}<\sff{str}>) \gtc. a {!} s (\msg \sff{auth}<\sff{bool}>) \gtc. X \gtc,
            \\
            \sff{quit} \gtc. c {!} a (\sff{quit}) \gtc. \tEnd
        \end{array}
    }
\]
The relative projections of $\gtc{G_{\sff{auth}}}$ are as follows:
\begin{align*}
    \gtc{G_{\sff{auth}}} \wrt (s,a)
    &= \mu X \rtc. (s {!} c) {!} a \rtBraces*{
        \begin{array}{@{}l@{}}
            \sff{login} \rtc. \tSkip \rtc. a {!} s (\msg \sff{auth}<\sff{bool}>) \rtc. X \rtc,
            \\
            \sff{quit} \rtc. \tSkip \rtc. \tEnd
        \end{array}
    }
    \\
    \gtc{G_{\sff{auth}}} \wrt (c,a)
    &= \mu X \rtc. (c {?} s) {!} a \rtBraces*{
        \begin{array}{@{}l@{}}
            \sff{login} \rtc. c {!} a (\msg \sff{passwd}<\sff{str}>) \rtc. \tSkip \rtc. X \rtc,
            \\
            \sff{quit} \rtc. c {!} a (\sff{quit}) \rtc. \tEnd
        \end{array}
    }
    \\
    \gtc{G_{\sff{auth}}} \wrt (s,c)
    &= \mu X \rtc. s {!} c \rtBraces*{
        \begin{array}{@{}l@{}}
            \sff{login} \rtc. \tSkip^2 \rtc. X \rtc,
            \\
            \sff{quit} \rtc. \tSkip \rtc. \tEnd
        \end{array}
    }
\end{align*}

\begin{figure}[p]
    \ifthesis\else \begin{mdframed} \fi
        \begin{align*}
            &\vdash {\rtr{c}}
            = \pRec X(\ci{\mu}{c}, \crt{c}{s}, \crt{c}{a}) ; \pBra* \crt{c}{s} > \left\{
                \begin{array}{@{}l@{}l@{}l@{}}
                    \sff{login}
                    & {}:{}
                    & \pSel* \ci{\mu}{c} < \sff{login} \cdot \pSel* \crt{c}{a} < \sff{login} \cdot \pIn \crt{c}{s}(u) ; \pOut* \ci{\mu}{c}[u'] \cdot
                    \\
                    &
                    & ( \pFwd [u<>u'] \| \pBra* \ci{\mu}{c} > \left\{
                        \begin{array}{@{}l@{}l@{}}
                            \sff{passwd}:{}
                            & \pSel* \crt{c}{a} < \sff{passwd} \cdot \pIn \ci{\mu}{c}(v) ; \pOut* \crt{c}{a}[v'] \cdot
                            \\
                            & ( \pFwd [v<>v'] \| \pCall X<\ci{\mu}{c}, \crt{c}{s}, \crt{c}{a}> )
                        \end{array}
                    \right\} ) ,
                    \\
                    \sff{quit}
                    & {}:{}
                    & \pSel* \ci{\mu}{c} < \sff{quit} \cdot \pSel* \crt{c}{a} < \sff{quit} \cdot \pIn \crt{c}{s}(w) ; \pOut* \ci{\mu}{c}[w'] \cdot
                    \\
                    &
                    & ( \pFwd [w<>w'] \| \pBra* \ci{\mu}{c} > \left\{
                        \begin{array}{@{}l@{}l@{}}
                            \sff{quit}:{}
                            & \pSel* \crt{c}{a} < \sff{quit} \cdot \pIn \ci{\mu}{c}(z) ; \pOut* \crt{c}{a}[z'] \cdot
                            \\
                            & ( \pFwd [z<>z'] \| \0 )
                        \end{array}
                    \right\} )
                \end{array}
            \right\}
            \\
            &\typInf \begin{array}[t]{@{}l@{}}
                \ci{\mu}{c}:
                \mu X . \oplus^{2} \left\{
                    \begin{array}{@{}l@{}l@{}}
                        \sff{login}
                        & {}: \bullet \tensor^{3} \&^{4} \{ \sff{passwd}: \bullet \parr^{5} X \},
                        \\
                        \sff{quit}
                        & {}: \bullet \tensor^{3} \&^{4} \{ \sff{quit}: \bullet \parr^{5} \bullet \}
                    \end{array}
                \right\} = \ol{(\gtc{G_{\sff{auth}}} \locProjD^0 c)},
                \\
                \crt{c}{s}:
                \mu X . \&^{1} \{ \sff{login}: \bullet \parr^{2} X, \sff{quit}: \bullet \parr^{2} \bullet \} = \relProp{c}{s}{0}{\gtc{G_{\sff{auth}}} \wrt (c,s)},
                \\
                \crt{c}{a}:
                \mu X . \oplus^{2} \left\{
                    \begin{array}{@{}l@{}l@{}}
                        \sff{login}
                        & {}: \oplus^{5} \{ \sff{passwd}: \bullet \tensor^{6} X \},
                        \\
                        \sff{quit}
                        & {}: \oplus^{5} \{ \sff{quit}: \bullet \tensor^{6} \bullet \}
                    \end{array}
                \right\} = \relProp{c}{a}{0}{\gtc{G_{\sff{auth}}} \wrt (c,a)}
            \end{array}
            \\[3pt] \cline{1-2}
            \\[-15pt]
            &\vdash \rtr{s}
            = \pRec X(\ci{\mu}{s}, \crt{s}{c}, \crt{s}{a}) ; \pBra* \ci{\mu}{s} > \left\{
                \begin{array}{@{}l@{}l@{}l@{}}
                    \sff{login}
                    & {}:{}
                    & \pSel* \crt{s}{c} < \sff{login} \cdot \pSel* \crt{s}{a} < \sff{login} \cdot \pIn \ci{\mu}{s}(u) ; \pOut* \crt{s}{c}[u'] \cdot
                    \\
                    &
                    & ( \pFwd [u<>u'] \| \pBra* \crt{s}{a} > \left\{
                        \begin{array}{@{}l@{}l@{}}
                            \sff{auth}:{}
                            & \pSel* \ci{\mu}{s} < \sff{auth} \cdot \pIn \crt{s}{a}(v) ; \pOut* \ci{\mu}{s}[v'] \cdot
                            \\
                            & ( \pFwd [v<>v'] \| \pCall X<\ci{\mu}{s}, \crt{s}{c}, \crt{s}{a}> )
                        \end{array}
                    \right\} ),
                    \\
                    \sff{quit}
                    & {}:{}
                    & \pSel* \crt{s}{c} < \sff{quit} \cdot \pSel* \crt{s}{a} < \sff{quit} \cdot \pIn \ci{\mu}{s}(v) ; \pOut* \crt{s}{c}[v'] \cdot
                    \\
                    &
                    & ( \pFwd [v<>v'] \| \0 )
                \end{array}
            \right\}
            \\
            &\typInf \begin{array}[t]{@{}l@{}}
                \ci{\mu}{s}:
                \mu X . \&^{0} \{ \sff{login}: \bullet \parr^{1} \oplus^{10} \{ \sff{auth}: \bullet \tensor^{11} X \}, \sff{quit}: \bullet \parr^{1} \bullet \} = \ol{(\gtc{G_{\sff{auth}}} \locProjD^0 s)},
                \\
                \crt{s}{c}:
                \mu X . \oplus^{1} \{ \sff{login}: \bullet \tensor^{2} X, \sff{quit}: \bullet \tensor^{2} \bullet \} = \relProp{s}{c}{0}{\gtc{G_{\sff{auth}}} \wrt (s,c)},
                \\
                \crt{s}{a}:
                \mu X . \oplus^{1} \{ \sff{login}: \&^{9} \{ \sff{auth}: \bullet \parr^{10} X \}, \sff{quit}: \bullet \} = \relProp{s}{a}{0}{\gtc{G_{\sff{auth}}} \wrt (s,a)}
            \end{array}
            \\[3pt] \cline{1-2}
            \\[-15pt]
            &\vdash \rtr{a}
            = \pRec X(\ci{\mu}{a}, \crt{a}{c}, \crt{a}{s}) ; \pBra* \crt{a}{s} > \left\{
                \begin{array}{@{}l@{}}
                    \sff{login}: \pSel* \ci{\mu}{a} < \sff{login} \cdot
                    \\
                    \pBra* \crt{a}{c} > \left\{
                        \begin{array}{@{}l@{}}
                            \sff{login}:{}
                            \\
                            \pBra* \crt{a}{c} > \left\{
                                \begin{array}{@{}l@{}}
                                    \sff{passwd}: \pSel* \ci{\mu}{a} < \sff{passwd} \cdot \pIn \crt{a}{c}(u) ; \pOut* \ci{\mu}{a}[u'] \cdot
                                    \\
                                    {} \left(
                                        \begin{array}{@{}l@{}}
                                            \pFwd [u<>u']
                                            \\
                                            {} \| \pBra* \ci{\mu}{a} > \left\{
                                                \begin{array}{@{}l@{}}
                                                    \sff{auth}:{}
                                                    \\
                                                    \pSel* \crt{a}{s} < \sff{auth} \cdot \pIn \ci{\mu}{a}(v) ; \pOut* \crt{a}{s}[v'] \cdot
                                                    \\
                                                    ( \pFwd [v<>v'] \| \pCall X<\ci{\mu}{a}, \crt{a}{c}, \crt{a}{s}> )
                                                \end{array}
                                            \right\}
                                        \end{array}
                                    \right)
                                \end{array}
                            \right\},
                            \\
                            \sff{quit}: \pAlarm(\ci{\mu}{a},\crt{a}{c},\crt{a}{s})
                        \end{array}
                    \right\},
                    \\
                    \sff{quit}: \pSel* \ci{\mu}{a} < \sff{quit} \cdot
                    \\
                    \pBra* \crt{a}{c} > \left\{
                        \begin{array}{@{}l@{}}
                            \sff{login}: \pAlarm(\ci{\mu}{a},\crt{a}{c},\crt{a}{s}),
                            \\
                            \sff{quit}: \begin{array}[t]{@{}l@{}}
                                \pBra* \crt{a}{c} > \{ \sff{quit}: \pSel* \ci{\mu}{a} < \sff{quit} \cdot \pIn \crt{a}{c}(w) ; \pOut* \ci{\mu}{a}[w'] \cdot
                                \\
                                ( \pFwd [w<>w'] \| \0 ) \}
                            \end{array}
                        \end{array}
                    \right\}
                \end{array}
            \right\}
            \\
            &\typInf \begin{array}[t]{l}
                \ci{\mu}{a}:
                \mu X . \oplus^{2} \left\{
                    \begin{array}{@{}l@{}l@{}}
                        \sff{login}
                        & {}: \oplus^{6} \{ \sff{passwd}: \bullet \tensor^{7} \&^{8} \{ \sff{auth}: \bullet \parr^{9} X \} \},
                        \\
                        \sff{quit}
                        & {}: \oplus^{6} \{ \sff{quit}: \bullet \tensor^{7} \bullet \}
                    \end{array}
                \right\} = \ol{(\gtc{G_{\sff{auth}}} \locProjD^0 a)},
                \\
                \crt{a}{c}:
                \mu X . \&^{2} \left\{
                    \begin{array}{@{}l@{}l@{}}
                        \sff{login}
                        & {}: \&^{5} \{ \sff{passwd}: \bullet \parr^{6} X \} ,
                        \\
                        \sff{quit}
                        & {}: \&^{5} \{ \sff{quit}: \bullet \parr^{6} \bullet \}
                    \end{array}
                \right\} = \relProp{a}{c}{0}{\gtc{G_{\sff{auth}}} \wrt (a,c)},
                \\
                \crt{a}{s}:
                \mu X . \&^{1} \{ \sff{login}: \oplus^{9} \{ \sff{auth}: \bullet \tensor^{10} X \} , \sff{quit}: \bullet \} = \relProp{a}{s}{0}{\gtc{G_{\sff{auth}}} \wrt (a,s)}
            \end{array}
        \end{align*}
    \ifthesis\else \end{mdframed} \fi

    \caption{Routers synthesized from $\gtc{G_{\sff{auth}}}$.}
    \label{f:mpstAPCP:GauthRouters}
\end{figure}

The typed routers synthesized from $\gtc{G_{\sff{auth}}}$ are given in \Cref{f:mpstAPCP:GauthRouters}.
Let us explain the behavior of $\rtr{a}$, the router of $a$.
$\rtr{a}$ is a recursive process on recursion variable $X$, using the endpoint for the implementation $\ci{\mu}{a}$ and the endpoint for the other routers $\crt{a}{c}$ and $\crt{a}{s}$ as context.
The initial message in $\gtc{G_{\sff{auth}}}$ from $s$ to $c$ is a dependency for $a$'s interactions with both~$s$ and $c$.
Therefore, the router first branches on the first dependency with $s$: a label received over $\crt{a}{s}$ (\sff{login} or \sff{quit}).
Let us detail the \sff{login} branch.
Here, the router sends \sff{login} over $\ci{\mu}{a}$.
Then, the router branches on the second dependency with $c$: a label received over~$\crt{a}{c}$ (again, \sff{login} or \sff{quit}).
\begin{itemize}
    \item
        In the second \sff{login} branch, the router receives the label \sff{passwd} over $\crt{a}{c}$, which it then sends over $\ci{\mu}{a}$.
        The router then receives an endpoint (the password) over $\crt{a}{c}$, which it forwards over $\ci{\mu}{a}$.
        Finally, the router receives the label \sff{auth} over $\ci{\mu}{a}$, which it sends over~$\crt{a}{s}$.
        Then, the router receives an endpoint (the authorization result) over $\ci{\mu}{a}$, which it forwards over~$\crt{a}{s}$.
        The router then recurses to the beginning of the loop on the recursion variable $X$, passing the endpoints $\ci{\mu}{a},\crt{a}{s},\crt{a}{c}$ as recursive context.
    \item
        In the \sff{quit} branch, the router is in an inconsistent state, because it has received a label over $\crt{a}{c}$ which does not concur with the label received over $\crt{a}{s}$.
        Hence, the router signals an alarm on its endpoints $\ci{\mu}{a},\crt{a}{s},\crt{a}{c}$.
\end{itemize}

Notice how the typing of the routers in \Cref{f:mpstAPCP:GauthRouters} follows \Cref{t:mpstAPCP:routerTypes}: for each $p \in \{c,s,a\}$, the endpoint $\ci{\mu}{p}$ is typed with local projection (\Cref{d:mpstAPCP:locProj}), and for each $q \in \{c,s,a\} \setminus \{p\}$ the endpoint $\crt{p}{q}$ is typed with relative projection (\Cref{d:mpstAPCP:relProj,d:mpstAPCP:relProp}).

\begin{figure}[t]
    \ifthesis\else \begin{mdframed} \fi
        \begin{align*}
            &\vdash \mdm{\{c,s,a\}}{\gtc{G_{\sff{auth}}}}
            \\
            &= \begin{array}[t]{@{}l@{}}
                \pRec X(\ci{\mu}{c}, \ci{\mu}{s}, \ci{\mu}{a}) ;
                \\
                \pBra* \ci{\mu}{s} > \left\{
                    \begin{array}{@{}l@{}l@{}l@{}}
                        \sff{login}
                        & {}:{}
                        & \pSel* \ci{\mu}{c} < \sff{login} \cdot \pSel* \ci{\mu}{a} < \sff{login} \cdot \pIn \ci{\mu}{s}(u) ; \pOut* \ci{\mu}{c}[u'] \cdot ( \pFwd [u<>u']
                        \\
                        &
                        & {} \| \pBra* \ci{\mu}{c} > \left\{
                            \begin{array}{@{}l@{}l@{}}
                                \sff{auth}:{}
                                & \pSel* \ci{\mu}{a} < \sff{passwd} \cdot \pIn \ci{\mu}{c}(v) ; \pOut* \ci{\mu}{a}[v'] \cdot
                                \\
                                & ( \pFwd [v<>v'] \| \pBra* \ci{\mu}{a} > \left\{
                                    \begin{array}{@{}l@{}l@{}}
                                        \sff{auth}:{}
                                        & \pSel* \ci{\mu}{s} < \sff{auth} \cdot \pIn \ci{\mu}{a}(w) ; \pOut* \ci{\mu}{s}[w'] \cdot
                                        \\
                                        & ( \pFwd [w<>w'] \| \pCall X<\ci{\mu}{c}, \ci{\mu}{s}, \ci{\mu}{a}> )
                                    \end{array}
                                \right\} )
                            \end{array}
                        \right\} ),
                        \\[6pt]
                        \sff{quit}
                        & {}:{}
                        & \pSel* \ci{\mu}{c} < \sff{quit} \cdot \pSel* \ci{\mu}{a} < \sff{quit} \cdot \pIn \ci{\mu}{s}(z) ; \pOut* \ci{\mu}{c}[z'] \cdot
                        \\
                        &
                        & ( \pFwd [z<>z'] \| \pBra* \ci{\mu}{c} > \{ \sff{quit}: \pSel* \ci{\mu}{a} < \sff{quit} \cdot \pIn \ci{\mu}{c}(y) ; \pOut* \ci{\mu}{a}[y'] \cdot ( \pFwd [y<>y'] \| \0 ) \} )
                    \end{array}
                \right\}
            \end{array}
            \\
            &\typInf \ci{\mu}{c}: \ol{\gtc{G_{\sff{auth}}} \locProjD^0 c}, \ci{\mu}{s}: \ol{\gtc{G_{\sff{auth}}} \locProjD^0 s}, \ci{\mu}{a}: \ol{\gtc{G_{\sff{auth}}} \locProjD^0 a}
        \end{align*}
    \ifthesis\else \end{mdframed} \fi

    \caption{Orchestrator synthesized from $\gtc{G_{\sff{auth}}}$ (cf.\ \Cref{d:mpstAPCP:medium}).}
    \label{f:mpstAPCP:GauthMedium}
\end{figure}

Consider again the participant implementations given in \Cref{x:mpstAPCP:impl}: $P$ implements the role of $c$, $Q$ the role of $s$, and $R$ the role of $a$.
Notice that the types of the channels of these processes coincide with relative projections:
\begin{align*}
    \emptyset \vdash P &\typInf \ci{c}{\mu}: \gtc{G_{\sff{auth}}} \locProjD^0 c
    &
    \emptyset \vdash Q &\typInf \ci{s}{\mu}: \gtc{G_{\sff{auth}}} \locProjD^0 s
    &
    \emptyset \vdash R &\typInf \ci{a}{\mu}: \gtc{G_{\sff{auth}}} \locProjD^0 a
\end{align*}
Let us explore how to compose these implementations with their respective routers.
The order of composition determines the network topology.
\begin{description}
    \item[Decentralized]
        By composing each router with their respective implementation, and then composing the resulting routed implementations, we obtain a decentralized topology:
        \[
            N_{\sff{auth}}^{\sff{decentralized}} \deq \begin{array}{@{}r@{}}
                \pRes{\crt{c}{s} \crt{s}{c}}
                \\
                \pRes{\crt{c}{a} \crt{a}{c}}
                \\
                \pRes{\crt{s}{a} \crt{a}{s}}
            \end{array} \left(
                \begin{array}{@{}l@{}l@{}l@{}}
                    & \pRes{\ci{\mu}{c} \ci{c}{\mu}}
                    & ( \rtr{c} \| P )
                    \\
                    {}\|{}
                    & \pRes{\ci{\mu}{s} \ci{s}{\mu}}
                    & ( \rtr{s} \| Q )
                    \\
                    {}\|{}
                    & \pRes{\ci{\mu}{a} \ci{a}{\mu}}
                    & ( \rtr{a} \| R )
                \end{array}
            \right)
        \]
        This composition is a network of routed implementations of $\gtc{G_{\sff{auth}}}$ (cf.\ \Cref{d:mpstAPCP:networks}), so \Cref{t:mpstAPCP:completeness,t:mpstAPCP:soundness,t:mpstAPCP:globalDf} apply: we have \mbox{$N_{\sff{auth}}^{\sff{decentralized}} \in \sys(\gtc{G_{\sff{auth}}})$}, so $N_{\sff{auth}}^{\sff{decentralized}}$ behaves as specified by $\gtc{G_{\sff{auth}}}$ and is deadlock-free.

    \item[Centralized]
        By first composing the routers, and then composing the connected routers with each implementation, we obtain a centralized topology:
        \[
            N_{\sff{auth}}^{\sff{centralized}} \deq \begin{array}{@{}c@{}}
                \pRes{\ci{\mu}{c} \ci{c}{\mu}}
                \\
                \pRes{\ci{\mu}{s} \ci{s}{\mu}}
                \\
                \pRes{\ci{\mu}{a} \ci{a}{\mu}}
            \end{array} \left(
                \begin{array}{@{}l@{}}
                    \begin{array}{@{}c@{}}
                        \pRes{\crt{c}{s} \crt{s}{c}}
                        \\
                        \pRes{\crt{c}{a} \crt{a}{c}}
                        \\
                        \pRes{\crt{s}{a} \crt{a}{s}}
                    \end{array} \left(
                        \begin{array}{@{}l@{}l@{}}
                            & \rtr{c}
                            \\
                            {}\|{}
                            & \rtr{s}
                            \\
                            {}\|{}
                            & \rtr{a}
                        \end{array}
                    \right) \begin{array}{@{}l@{}}
                        {} \| P
                        \\
                        {} \| Q
                        \\
                        {} \| R
                    \end{array}
                \end{array}
            \right)
        \]
        Note that the composition of routers is a \emph{hub of routers} (\Cref{d:mpstAPCP:hub}).
        Consider the composition of $P$, $Q$ and $R$ with the orchestrator of $\gtc{G_{\sff{auth}}}$ (given in \Cref{f:mpstAPCP:GauthMedium}):
        \[
            N_{\sff{auth}}^{\sff{orchestrator}} \deq \begin{array}{@{}c@{}}
                \pRes{\ci{\mu}{c} \ci{c}{\mu}}
                \\
                \pRes{\ci{\mu}{s} \ci{s}{\mu}}
                \\
                \pRes{\ci{\mu}{a} \ci{a}{\mu}}
            \end{array} \left(
                \mdm{\{c,s,a\}}{\gtc{G_{\sff{auth}}}} \begin{array}{@{}l@{}}
                    {} \| P
                    \\
                    {} \| Q
                    \\
                    {} \| R
                \end{array}
            \right)
        \]
        By \Cref{t:mpstAPCP:mediumBisim}, the hub of routers and the orchestrator of $\gtc{G_{\sff{auth}}}$ are weakly bisimilar (\Cref{d:mpstAPCP:weakBisim}).
        Hence, $N_{\sff{auth}}^{\sff{centralized}}$ and $N_{\sff{auth}}^{\sff{orchestrator}}$ behave the same.
\end{description}
Since each of $N_{\sff{auth}}^{\sff{top}}$ with $\sff{top} \in \{\sff{decentralized}, \sff{centralized}, \sff{orchestrator}\}$ is typable in empty contexts, by \Cref{t:mpstAPCP:globalDf}, each of these compositions is deadlock-free.
Moreover, $N_{\sff{auth}}^{\sff{decentralized}}$ and $N_{\sff{auth}}^{\sff{centralized}}$ are structurally congruent, so, by \Cref{t:mpstAPCP:completeness,t:mpstAPCP:soundness}, they behave as prescribed by $\gtc{G_{\sff{auth}}}$.
Finally, by \Cref{t:mpstAPCP:mediumBisim}, $N_{\sff{auth}}^{\sff{centralized}}$ and $N_{\sff{auth}}^{\sff{orchestrator}}$ are bisimilar, and so $N_{\sff{auth}}^{\sff{orchestrator}}$ also behaves as prescribed by $\gtc{G_{\sff{auth}}}$.

\section{Related Work}
\label{s:relwork}


\paragraph{Types for deadlock-freedom.}

Our decentralized analysis of global types is related to type systems that ensure deadlock-freedom for multiparty sessions with delegation and interleaving~\cite{conf/concur/BettiniCDLDY08,conf/coordination/PadovaniVV14,journal/mscs/CoppoDYP16}.
Unlike these works, we rely on a type system for \emph{binary} sessions which is simple and enables an expressive analysis of global types.
Coppo \etal~\cite{conf/concur/BettiniCDLDY08,conf/coordination/CoppoDPY13,journal/mscs/CoppoDYP16} give type systems for multiparty protocols, with asynchrony and support for interleaved sessions by tracking of mutual dependencies between them; as per Toninho and Yoshida~\cite{journal/toplas/ToninhoY18}, our example in \Cref{s:mpstAPCP:examples:intrl} is typable in \APCP but not in their system.
Padovani \etal~\cite{conf/coordination/PadovaniVV14} develop a type system that enforces liveness properties for multiparty sessions, defined on top of a $\pi$-calculus with labeled communication.
Rather than global types, their type structure follows approaches based on \emph{conversation types}~\cite{journal/tcs/CairesV10}.
Toninho and Yoshida~\cite{journal/toplas/ToninhoY18} analyze binary sessions, leveraging on deadlock-freedom results for multiparty sessions to extend Wadler's CLL~\cite{conf/icfp/Wadler12} with cyclic networks.
Their process language is synchronous and uses replication rather than recursion.
We note that their Examples~6.8 and~6.9 can be typed in \APCP (cf.\ \Cref{s:mpstAPCP:examples:intrl});
a detailed comparison between their extended CLL and \APCP is interesting future work.

\paragraph{\MPSTs and binary analyses of global types.}

\ifmpstAPCP\else
There are many works on \MPST and their integration into programming languages; see~\cite{journal/csur/HuttelLVCCDMPRT16,journal/ftpl/AnconaBB0CDGGGH16} for surveys.
Triggered by flawed proofs of type safety and limitations of usual theories, Scalas and Yoshida~\cite{conf/popl/ScalasY19} define a meta-framework of multiparty protocols based on local types, without global types and projection. Their work has been a source of inspiration for our developments; we address similar issues by adopting relative types, instead of cutting ties with global types.
\fi

As already mentioned, Caires and Pérez~\cite{conf/forte/CairesP16} and Carbone \etal~\cite{conf/concur/CarboneMSY15,conf/concur/CarboneLMSW16,journal/ai/CarboneMSY17} reduce the analysis of global types to binary session type systems based on intuitionistic and classical linear logic, respectively.
Our routers strictly generalize the centralized mediums of Caires and Pérez (cf.\ \Cref{s:mpstAPCP:analysis:mediums}).
We substantially improve over the expressivity of the decentralized approach of Carbone \etal based on coherence, but reliant on encodings into centralized arbiters; for instance, their approach does not support the example from Toninho and Yoshida~\cite{journal/toplas/ToninhoY18} we discuss in \Cref{s:mpstAPCP:examples:intrl}.
Also, Caires and Pérez support neither recursive global types nor asynchronous communication, and neither do Carbone \etal.

Scalas \etal~\cite{conf/ecoop/ScalasDHY17} leverage on an encoding of binary session types into \emph{linear types}~\cite{conf/ppdp/DardhaGS12,journal/toplas/KobayashiPT99} to reduce multiparty sessions to processes typable with linear types, with applications in Scala programming.
Their analysis is decentralized but covers processes with synchronous communication only; also, their deadlock-freedom result is limited with respect to ours: it does not support interleaving, such as in the example in \Cref{s:mpstAPCP:examples:intrl}.

\paragraph{Monitoring through \MPST.}

Our work and the works discussed so far all consider the verification of implementations of multiparty protocols through static type checking.
Bocchi \etal~\cite{journal/tcs/BocchiCDHY17} use a \emph{dynamic} approach: communication between implementations is enacted by \emph{monitors}, which are derived from the global type to prevent protocol violations.
In their approach, Bocchi \etal rely on the traditional workflow for \MPST: projection onto binary session types based on the merge operation.
Interestingly, Bocchi \etal's semantics relies on \emph{routing}, which is similar in spirit, but not in details, to our routers: their routing approach abstracts away from the actual network structure, while our routers enable the concrete realization of a decentralized network structure.
We also note that Bocchi \etal's monitors, based on finite state machines, live on the level of semantics, while our routers, $\pi$-calculus processes, live on the same level as implementations.
The theory by Bocchi \etal has resulted in the development of tools for a practical application of monitoring in Python~\cite{journal/fmsd/DemangeonHHNY15}, including an extension to real-time systems~\cite{journal/fac/NeykovaBY17}.

\paragraph{Other approaches to multiparty protocols.}

In a broader context, Message Sequence Charts (\glspl{MSC}) provide graphical specifications of multiparty protocols.
Alur \etal~\cite{journal/tcs/AlurEY05} and Abdallah \etal~\cite{journal/ssm/AbdallahHJ15} study the decidability of model-checking properties such as implementability of MSC Graphs and High-level MSCs (\glspl{HMSC}) as Communicating \FSMs (\glspl{CFSM}).
\ifmpstAPCP
Genest \etal~\cite{journal/jcss/GenestMSZ06} study the synthesis of implementations of \glspl{HMSC} as \glspl{CFSM}; as we do, they use extra synchronization messages in some cases.
We follow an entirely different research strand: our analysis is type-based and targets well-formed global types that are implementable by design.
\fi
We note that the decidability of key notions for \MPST (such as well-formedness and typability) has been addressed in~\cite{journal/acm/HondaYC16}.

Collaboration diagrams are another visual model for communicating processes (see, e.g.~\cite{journal/soca/BultanF08}).
Salaün \etal~\cite{journal/tsc/SalaunBR12} encode collaboration diagrams into the LOTOS process algebra~\cite{report/Brinksma89} to enable model-checking~\cite{conf/cav/GaravelMLS07}, realizability checks for synchronous and asynchronous communication, and synthesis of participant implementations.
\ifmpstAPCP
Their implementation synthesis is reminiscent of our router synthesis, and also adds extra synchronization messages to realize otherwise unrealizable protocols with non-local choices.
\fi

\section{Conclusion}
\label{s:concl}

We have developed a new analysis of multiparty protocols specified as global types.
As a distinguishing feature, our analysis accounts for multiparty protocols implemented by arbitrary process networks, which can be centralized (as in orchestration-based approaches) but also decentralized (as in choreography-based approaches).
Another salient feature is that we can ensure both protocol conformance (protocol fidelity, communication safety) and deadlock-freedom, which is notoriously hard to establish for protocols/implementations involving delegation and interleaving.
To this end, we have considered asynchronous process implementations in \APCP, the typed process language that we introduced in~\cite{conf/ice/vdHeuvelP21} (\Cref{c:APCP}).
Our analysis enables the transference of correctness properties from \APCP to multiparty protocols.
We have illustrated these features using the authorization protocol $\gtc{G_{\sff{auth}}}$ adapted from Scalas and Yoshida~\cite{conf/popl/ScalasY19} as a running example; additional examples further justify how our approach improves over previous analyses (cf.\ \Cref{s:mpstAPCP:examples}).

Our analysis of multiparty protocols rests upon three key innovations:
\emph{routers}, which enable global type analysis as decentralized networks;
\emph{relative types} that capture protocols between pairs of participants;
\emph{relative projection}, which admits global types with non-local choices.
In our opinion, these notions are interesting on their own.
In particular, relative types shed new light on more expressive protocol specifications than usual \MPST, which are tied to notions of local types and merge/subtyping.

There are several interesting avenues for {future work}.
\ifthesis
We would
\else
Comparing relative and merge-based well-formedness would continue the tread of new projections of global types (cf.\ \appref{a:compare} for initial findings).
We would also
\fi
like to develop a type system based on relative types, integrating the logic of routers into a static type checking that ensures deadlock-freedom for processes.
\ifthesis\else
Finally, we are interested in developing practical tool support based on our findings.
For this latter point, following \cite{conf/popl/JiaGP16}, we would  like  to first formalize a theory of runtime monitoring based on routers, which can already be seen as an elementary form of \emph{choreographed monitoring} (cf.\ \cite{book/FrancalanzaPS18}).
\fi

\paragraph{Acknowledgments}

We are grateful to the anonymous reviewers for their constructive feedback and suggestions, which were enormously helpful to improve the presentation. We are also grateful to Jules Jacobs for valuable comments and suggestions.

This research has been supported by the Dutch Research Council (NWO) under project No. 016.Vidi.189.046 (Unifying Correctness for Communicating Software).

\phantomsection
\addcontentsline{toc}{section}{References}
\bibliographystyle{plainurl}
\bibliography{refs}


\clearpage
\appendix

\section{Comparing Merge-based Well-for\-med\-ness and Relative Well-for\-med\-ness}
\label{a:compare}

It is instructive to examine how the notion of well-formed global types induced by our relative projection compares to \emph{merge-based} well-formedness~\cite{journal/lmcs/HuBYD12,book/CarboneYH09}.

To define merge-based well-formedness, we define local types and projection of global types onto local types.
Local types express one participant's perspective of a global protocol.
\ifthesis\else
Although $\tSkip$ is not part of standard definitions of local types, we include it to enable a fair comparison with relative types.
\fi

\begin{definition}[Local Types]
    \label{d:relProj:localTypes}
    \emph{Local types} $L$ are defined as follows, where the $S$ are the message types from \Cref{d:relProj:globalTypes}:
    \ifthesis
    \[
        L
        \begin{array}[t]{@{}l@{}lr@{\kern2em}l@{\kern1em}lr@{}}
            {} \deq {} &
            p {!} q \ltBraces{ \msg i<S> . L }_{i \in I} & \text{send}
            & \sepr &
            \tEnd & \text{end}
            \\ \sepr* &
            p {?} q \ltBraces{ \msg i<S> . L }_{i \in I} & \text{receive}
            & \sepr &
            \mu X . L \sepr X & \text{recursion}
        \end{array}
    \]
    \else
    \[
        L
        \begin{array}[t]{@{}l@{}lr@{\kern2em}l@{\kern1em}lr@{}}
            {} \deq {} &
            p {!} q \ltBraces{ \msg i<S> . L }_{i \in I} & \text{send}
            & \sepr &
            \tEnd & \text{end}
            \\ \sepr* &
            p {?} q \ltBraces{ \msg i<S> . L }_{i \in I} & \text{receive}
            & \sepr &
            \mu X . L \sepr X & \text{recursion}
            \\
            &
            &
            & \sepr &
            \tSkip . L & \text{skip}
        \end{array}
    \]
    \fi
\end{definition}

\noindent
The local types $p {!} q \ltBraces{ \msg i<S_i> . L_i }_{i \in I}$ and $p {?} q \ltBraces{ \msg i<S_i> . L_i }_{i \in I}$ represent sending a choice from and receiving a choice from $q$, respectively.
\ifthesis
All of $\tEnd$, $\mu X . L$, and $X$ are just as before.
\else
All of $\tEnd$, $\mu X . L$, $X$, and $\tSkip$ are just as before.
\fi

Instead of dependency messages, local projection relies on \emph{merge}.
Intuitively, merge combines overlapping but not necessarily identical receives.
This is a main difference with respect to our relative projection.

\begin{definition}[Merge of Local Types]
    \label{d:relProj:merge}
    For local types $L_1$ and $L_2$, we define $L_1 \lmerge L_2$ as the \emph{merge} of $L_1$ and $L_2$:
    \ifthesis
    \begin{align*}
        \mu X . L_1 \lmerge \mu X . L_2
        &\deq \mu X . (L_1 \lmerge L_2)
        &
        X \lmerge X
        &\deq X
        \\
        p {!} q \ltBraces{ \msg i<S_i> . L_i }_{i \in I} \lmerge p {!} q \ltBraces{ \msg i<S_i> . L_i }_{i \in I} &\deq p {!} q \ltBraces{ \msg i<S_i> . L_i }_{i \in I}
        &
        \tEnd \lmerge \tEnd
        &\deq \tEnd
        \\
        p {?} q \ltBraces{ \msg i<S_i> . L_i }_{i \in I} \lmerge p {?} q \ltBraces{ \msg j<S'_j> . L'_j }_{j \in J}
        &\deq p {?} q \left(
            \begin{array}{@{}l@{}l@{}}
                & \ltBraces{ \msg i<S_i> . L_i }_{i \in I \setminus J}
                \\
                {}\cup{}
                & \ltBraces{ \msg j<S'_j> . L'_j }_{j \in J \setminus I}
                \\
                {}\cup{}
                & \ltBraces{ \msg k<S_k \lmerge S'_k> . (L_k \lmerge L'_k) }_{k \in I \cap J}
            \end{array}
        \right)
        \span\span
    \end{align*}
    \else
    \begin{align*}
        \tSkip . L_1 \lmerge \tSkip . L_2
        &\deq L_1 \lmerge L_2
        &
        \tEnd \lmerge \tEnd
        &\deq \tEnd
        \\
        \mu X . L_1 \lmerge \mu X . L_2
        &\deq \mu X . (L_1 \lmerge L_2)
        &
        X \lmerge X
        &\deq X
        \\
        p {!} q \ltBraces{ \msg i<S_i> . L_i }_{i \in I} \lmerge p {!} q \ltBraces{ \msg i<S_i> . L_i }_{i \in I} &\deq p {!} q \ltBraces{ \msg i<S_i> . L_i }_{i \in I}
        \\
        p {?} q \ltBraces{ \msg i<S_i> . L_i }_{i \in I} \lmerge p {?} q \ltBraces{ \msg j<S'_j> . L'_j }_{j \in J}
        &\deq p {?} q \left(
            \begin{array}{@{}l@{}l@{}}
                & \ltBraces{ \msg i<S_i> . L_i }_{i \in I \setminus J}
                \\
                {}\cup{}
                & \ltBraces{ \msg j<S'_j> . L'_j }_{j \in J \setminus I}
                \\
                {}\cup{}
                & \ltBraces{ \msg k<S_k \lmerge S'_k> . (L_k \lmerge L'_k) }_{k \in I \cap J}
            \end{array}
        \right)
        \span\span
    \end{align*}
    \fi
    The merge between message types $S_1 \lmerge S_2$ corresponds to the identity function.
    If the local types do not match the above definition, their merge is undefined.
\end{definition}

We can now define local projection based on merge:

\ifthesis
\begin{figure}[t]
    \vspace{-1em}

    \begin{align*}
        \tEnd \locProj p
        &\deq \tEnd
        \\
        (\mu X \gtc. \gtc{G}) \locProj p
        &\deq \begin{cases}
            \tEnd
            & \text{if $\gtc{G} \locProj p = \tEnd$ or $\gtc{G} \locProj p = X$}
            \\
            \mu X . (\gtc{G} \locProj p)
            & \text{otherwise}
        \end{cases}
        &
        X \locProj p
        &\deq X
        \\
        (r {!} s \gtBraces{ \msg i<U_i> \gtc. \gtc{G_i} }_{i \in I}) \locProj p
        &\deq \begin{cases}
            p {?} r \ltBraces{ \msg i<U_i> . (\gtc{G_i} \locProj p) }_{i \in I}
            & \text{if $p = s$}
            \\
            p {!} s \ltBraces{ \msg i<U_i> . (\gtc{G_i} \locProj p) }_{i \in I}
            & \text{if $p = r$}
            \\
            \bigmerge_{i \in I} (\gtc{G_i} \locProj p)
            & \text{otherwise}
        \end{cases}
    \end{align*}

    \vspace{-1em}

    \caption{Merge-based local projection (cf.\ \Cref{d:relProj:localProj}).}
    \label{f:relProj:localProj}
\end{figure}
\fi

\begin{definition}[Merge-based Local Projection]
    \label{d:relProj:localProj}
    For global type $\gtc{G}$ and participant $p$, \ifthesis \Cref{f:relProj:localProj} defines \else we define \fi the \emph{merge-based local projection} of $\gtc{G}$ under $p$, denoted $\gtc{G} \locProj p$\ifthesis.\else, as follows:\fi
    \ifthesis \else
    \begin{align*}
        \tEnd \locProj p
        &\deq \tEnd
        &
        (\tSkip \gtc. \gtc{G}) \locProj p
        &\deq \tSkip . (\gtc{G} \locProj p)
        &
        X \locProj p
        &\deq X
        \\
        (\mu X \gtc. \gtc{G}) \locProj p
        &\deq \begin{cases}
            \tEnd
            & \text{if $\gtc{G} \locProj p = \tSkip^\ast . \tEnd$ or $\gtc{G} \locProj p = \tSkip^\ast . X$}
            \\
            \mu X . (\gtc{G} \locProj p)
            & \text{otherwise}
        \end{cases}
        \span\span\span\span
        \\
        (r {!} s \gtBraces{ \msg i<U_i> \gtc. \gtc{G_i} }_{i \in I}) \locProj p
        &\deq \begin{cases}
            p {?} r \ltBraces{ \msg i<U_i> . (\gtc{G_i} \locProj p) }_{i \in I}
            & \text{if $p = s$}
            \\
            p {!} s \ltBraces{ \msg i<U_i> . (\gtc{G_i} \locProj p) }_{i \in I}
            & \text{if $p = r$}
            \\
            \tSkip . (\bigmerge_{i \in I} (\gtc{G_i} \locProj p))
            & \text{otherwise}
        \end{cases}
        \span\span\span\span
    \end{align*}
    \fi
\end{definition}

\begin{definition}[Merge Well-formedness]
    \label{d:relProj:mwf}
    A global type $\gtc{G}$ is \emph{merge well-formed} if, for every $p \in \prt(\gtc{G})$, the merge-based local projection $G \locProj p$ is defined.
\end{definition}

The classes of relative and merge well-formed global types are \emph{incomparable}: many protocols are supported by both methods of projection, but some only by one; we explore the differences next.
It may be surprising that our new approach does not subsume the traditional.
Though outside of the scope of this chapter, in future work we plan to
Possible approaches are to partially integrate merge in relative projection, or to expand the detection of non-local choices beyond triplets of participants (cf.\ \Cref{x:relProj:mergeNotRel}).

\subsection{Relative Well-formed, Not Merge Well-formed}
\label{s:relProj:comp:rwfNotMwf}

The merge of local types with outgoing messages of different labels is undefined.
Therefore, if a global type has communications, e.g., from $s$ to $a$ with different labels across branches of a prior communication between $b$ and $a$, the global type is not merge well-formed.
In contrast, such global types can be relative well-formed, because the prior communication may induce a dependency.
Similarly, global types with communications with different participants across branches of a prior communication are never merge well-formed, but may be relative well-formed.
The following example illustrates a global type with messages of different labels across branches of a prior communication:

\begin{example}
    We give an adaptation of the two-buyer-seller protocol in which Seller ($s$) tells Alice ($a$) to pay or not, depending on whether Bob ($b$) tells $a$ to buy or not.
    \[
        \gtc{G_{\sff{rwf}}} \deq b {!} a \gtBraces*{
            \begin{array}{@{}l@{}}
                \sff{ok} \gtc. s {!} a (\msg \sff{pay}<\sff{int}>) \gtc. \tEnd \gtc,
                \\
                \sff{cancel} \gtc. s {!} a (\sff{cancel}) \gtc. \tEnd
            \end{array}
        }
    \]
    This protocol is relative well-formed, as the relative projections onto every pair of participants is defined.
    Notice how there is a dependency in the relative projection onto $s$ and $a$:
    \[
        \gtc{G_{\sff{rwf}}} \wrt (s,a) = (a {?} b) {!} s \rtBraces*{
            \begin{array}{@{}l@{}}
                \sff{ok} \rtc. s {!} a (\msg \sff{pay}<\sff{int}>) \rtc. \tEnd \rtc,
                \\
                \sff{cancel} \rtc. s {!} a (\sff{cancel}) \rtc. \tEnd
            \end{array}
        }
    \]
    However, we do not have merge well-formedness: the merge-based local projection onto~$s$ is not defined:
    \[
        \gtc{G_{\sff{rwf}}} \locProj s = \ifthesis\else\tSkip . (\fi s {!} a (\msg \sff{pay}<\sff{int}>) . \tEnd \lmerge s {!} a (\sff{cancel}) . \tEnd \ifthesis\else)\fi
        \tag*{\qedhere}
    \]
\end{example}

\subsection{Merge Well-formed, Not Relative Well-formed}

For a communication between, e.g., $a$ and $b$, to induce a dependency for subsequent communications between other participants, at least one of $a$ and $b$ must be involved.
Therefore, global types where communications with participants other than $a$ and $b$ have different labels across branches of a prior communication between $a$ and $b$ are never relative well-formed.
In contrast, merge can combine the reception of different labels, so such global types may be merge well-formed---as long as the sender is aware of which branch has been taken before.
The following example demonstrates such a situation, and explains how such global types can be modified to be relative well-formed:

\begin{example}
    \label{x:relProj:mergeNotRel}
    Consider a variant of the two-buyer-seller protocol in which Seller ($s$) invokes a new participant, Mail-service ($m$), to deliver the requested product.
    In the following global type, Bob ($b$) tells Alice ($a$) of its decision to buy or not, after which $b$ sends the same choice to $s$, who then either invokes $m$ to deliver the product or not:
    \[
        \gtc{G_{\sff{mwf}}} \deq b {!} a \gtBraces*{
            \begin{array}{@{}l@{}}
                \sff{ok} \gtc. b {!} s (\sff{ok}) \gtc. s {!} m (\msg \sff{deliver}<\sff{str}>) \gtc. \tEnd \gtc,
                \\
                \sff{quit} \gtc. b {!} s (\sff{quit}) \gtc. s {!} m (\sff{quit}) \gtc. \tEnd
            \end{array}
        }
    \]
    $\gtc{G_{\sff{mwf}}}$ is merge well-formed: the merge-based local projections onto all participants are defined.
    Notice how the two different messages from $s$ are merged in the merge-based local projection onto $m$:
    \[
        \gtc{G_{\sff{mwf}}} \locProj m = \ifthesis\else \tSkip^2 .\fi m {?} s \ltBraces*{
            \begin{array}{@{}l@{}}
                \msg \sff{deliver}<\sff{str}> . \tEnd ,
                \\
                \sff{quit} . \tEnd
            \end{array}
        }
    \]
    $\gtc{G_{\sff{mwf}}}$ is not relative well-formed: the relative projection onto $s$ and $m$ is not defined.
    The initial exchange between $b$ and $a$ cannot induce a dependency, since neither of $s$ and $m$ is involved.
    Hence, the relative projections of both branches must be identical, but they are not:
    \[
        \ifthesis\else\tSkip \rtc.\fi s {!} m (\msg \sff{deliver}<\sff{str}>) \rtc. \tEnd \neq \ifthesis\else\tSkip \rtc.\fi s {!} m (\sff{quit}) \rtc. \tEnd
    \]

    We recover relative well-formedness by modifying $\gtc{G_{\sff{mwf}}}$: we give $s$ the same options to send to $m$ in both branches of the initial communication:
    \[
        \gtc{G'_{\sff{mwf}}} \deq b {!} a \gtBraces*{
            \begin{array}{@{}l@{}}
                \sff{ok} \gtc. b {!} s (\sff{ok}) \gtc. s {!} m \gtBraces{ \msg \sff{deliver}<\sff{str}> \gtc. \tEnd \gtc, \sff{quit} \gtc. \tEnd } \gtc,
                \\
                \sff{quit} \gtc. b {!} s (\sff{quit}) \gtc. s {!} m \gtBraces{ \msg \sff{deliver}<\sff{str}> \gtc. \tEnd \gtc, \sff{quit} \gtc. \tEnd }
            \end{array}
        }
    \]
    The new protocol is still merge well-formed, but it is now relative well-formed too; the relative projection onto $s$ and $m$ is defined:
    \[
        \gtc{G'_{\sff{mwf}}} \wrt (s,m) = \ifthesis\else\tSkip^2 \rtc.\fi s {!} m \rtBraces*{
            \begin{array}{@{}l@{}}
                \msg \sff{deliver}<\sff{address}> \rtc. \tEnd \rtc,
                \\
                \sff{quit} \rtc. \tEnd
            \end{array}
        }
    \]
    This modification may not be ideal, though, because $s$ can quit the protocol even if $b$ has ok'ed the transaction, and that $s$ can still invoke a delivery even if $b$ has quit the transaction.
\end{example}



{\def\addcontentsline#1#2#3{}\maketitle}

\vspace{-5mm}
\begin{abstract}
    We study functional and concurrent calculi with \nond, along with  {type systems} to control resources based on  {linearity}.
The interplay between \nond and linearity is delicate: careless handling of branches can discard resources meant to be used exactly once.
Here we go beyond prior work by considering \nond in its standard sense: once a branch is selected, the rest are discarded.
Our technical contributions are three-fold.
    First, we introduce a $\pi$-calculus with non-de\-ter\-ministic choice, governed by session types.
    Second, we introduce a resource \lamcalc, governed by intersection types, in which \nond concerns fetching of resources from bags.
    Finally, we connect our two typed non-deterministic calculi via a correct translation.
    %
\end{abstract}

\section{Introduction}
\label{s:intro}

In this \ifthesis{}chapter\else{}paper\fi, we present new formulations of typed programming calculi with \emph{non-determinism}.
Non-determinism is a classical ingredient of models of computation, which brings flexibility and generality in specifications.
In process calculi such as CCS and the \picalc, one source of non-determinism is {choice}, which is typically \emph{\nconf}: that is, given $P + Q$, we have either $P + Q \redd P$ or $P + Q \redd Q$.
Thus, committing to a branch entails discarding the rest.

We are interested in non-determinism as a way of increasing the expressivity of typed calculi in which resource control is based on \emph{linearity}.
The interplay between non-determinism and linearity is delicate: a careless discarding of branches can jeopardize resources meant to be used exactly once.
\ifthesis We
\else On the concurrent side, we \fi
consider the \picalc, a general computational model in which concurrency, higher-order computation, and objects can be suitably expressed~\cite{book/SangiorgiW03}.
We focus on variants of the \picalc with \emph{session types}, in which linear logic principles ensure communication correctness~\cite{conf/concur/Honda93,conf/esop/HondaVK98}.
Here, the resources are names that perform session protocols; they can be \emph{unrestricted} (used multiple times) and \emph{linear} (used exactly once).
To properly control resources, \nconf non-determinism is confined to unrestricted names; linear names can only perform deterministic choices.

In this context, looking into \emph{confluent} forms of \nond can be appealing.
Intuitively, such formulations allow all branches to proceed independently: given $P_1 \redd Q_1$ and $P_2 \redd Q_2$, then $P_1 + P_2 \redd Q_1 + P_2$ and $P_1 + P_2 \redd  P_1 + Q_2$.
Because confluent \nond does not discard branches, it is compatible with a resource-conscious view of computation.

\Conf \nond has been studied mostly in the functional setting; it is present, e.g., in Pagani and Ronchi della Rocca's resource \lamcalc~\cite{conf/fossacs/PaganiR10} and in Ehrhard and Regnier's differential \lamcalc~\cite{journal/tcs/EhrhardR03}.
In~\cite{conf/fossacs/PaganiR10}, non-determinism resides in the application of a term $M$ to a \emph{bag} of available resources $C$; a $\beta$-reduction applies $M$ to a resource \emph{non-deterministically fetched} from $C$.
\Conf \nondt choice is also present in the session-typed \picalc by Caires and Pérez~\cite{conf/esop/CairesP17}, where it expresses a choice between different implementations of the same session protocols, which are all \emph{non-deterministically available}---they may be available but may also \emph{fail}.
In their work, a Curry-Howard correspondence between linear logic and session types (`\emph{propositions-as-sessions}'~\cite{conf/concur/CairesP10,conf/icfp/Wadler12}) ensures confluence, protocol fidelity, and deadlock-freedom.
Paulus et al.~\cite{conf/fscd/PaulusNP21} relate functional and concurrent calculi with confluent \nond: they give a translation of a resource \lamcalc into the session \picalc from~\cite{conf/esop/CairesP17}, in the style of Milner's `\emph{functions-as-pro\-ce\-sses}'~\cite{conf/icalp/Milner90}.

Although results involving confluent \nond are most significant, usual (\nconf) \nond remains of undiscussed convenience in formal modeling; consider, e.g., specifications of distributed protocols~\cite{conf/express/BergerH03,conf/concur/NestmannFM03} in which commitment is essential.
Indeed, \nconf \nondt choice is commonplace in verification frameworks such as mCRL2~\cite{book/GrooteM14}.
It is also relevant in functional calculi; a well-known framework is De'Liguoro and Piperno's (untyped) non-deterministic \lamcalc~\cite{journal/ic/DeliguoroP95} (see also~\cite{thesis/DezaniCiancaglini96} and references therein).

\paragraph{Motivating example.}

To illustrate the distinction between confluent and non-confluent non-determinism, we consider an example adapted from~\cite{conf/esop/CairesP17}: a movie server that offers a choice between buying a movie or watching its trailer.
In \clpi, the typed \picalc that we present here, this server can be specified as follows:
\[
    \sff{Server}_s = \pBra* s > \left\{
        \begin{array}{@{}l@{}}
            \sff{buy} : \pIn s(\textit{title}) ; \pIn s(\textit{paym}) ; \pOut* s[\texttt{movie}] ; \pClose s[]
            ~,~
            \\
            \sff{peek} : \pIn s(\textit{title}) ; \pOut* s[\texttt{trailer}] ;  \pClose s[]
        \end{array}
    \right\}
    \begin{array}{@{}l@{}}
        -\sff{Server}_s^{\sff{buy}}
        \\
        -\sff{Server}_s^{\sff{peek}}
    \end{array}
\]
where $\pIn s(-)$ and $\pOut* s[-]$ denote receive and send prefixes on a name $s$, respectively, and `\texttt{movie}' and `\texttt{trailer}' denote references to primitive data.
Also, the free names of a process are denoted with subscripts.
Process $\sff{Server}_s$ offers a choice on~$s$ ($\pBra* s > \{-\}$) between labels \sff{buy} and \sff{peek}.
If \sff{buy} is received, process $\sff{Server}_s^{\sff{buy}}$ is launched: it receives the movie's title and a payment method, sends the movie, and closes the session on~$s$ ($\pClose s[]$).
If \sff{peek} is received, it proceeds as $\sff{Server}_s^{\sff{peek}}$: the server receives the title, sends the trailer, and closes the session.

Using the non-deterministic choice operator of \clpi, denoted `$\nd$', we can specify the process $\sff{Alice}_s$, a client interested in the movie `Jaws' but undecided about buying the film or just watching its trailer for free:
\[
    \sff{Alice}_s \deq \begin{array}[t]{@{}r@{}lll@{}}
        &
        \pSel* s < \sff{buy} ;
        &
        \pOut* s[\texttt{Jaws}] ; \pOut* s[\texttt{mcard}] ; \pIn s(\textit{movie}) ; \pWait s() ; \0
        &
        -\sff{Alice}_s^{\sff{buy}}
        \\
        {} \nd {} &
        \pSel* s < \sff{peek} ;
        &
        \pOut* s[\texttt{Jaws}] ; \pIn s(\textit{trailer}) ; \pWait s() ; \0
        &
        -\sff{Alice}_s^{\sff{peek}}
    \end{array}
\]
If $\sff{Alice}_s$ selects on $s$  the label $\sff{buy}$ ($\pSel* s < \sff{buy}$), process $\sff{Alice}_s^{\sff{buy}}$ is launched: she sends title and payment method, receives the movie, waits for the session to close ($\pWait s()$), and then terminates~($\0$).
If she selects \sff{peek}, she proceeds as $\sff{Alice}_s^{\sff{peek}}$: she sends a title, receives the trailer, waits for the session to close, and terminates.
Process $\sff{Sys} \deq \pRes{s} ( \sff{Server}_s \| \sff{Alice}_s )$ denotes the composition of client and server, connected  along $s$.
Our semantics for \clpi, denoted $\reddLazy$, enforces non-confluent non-determinism, as $\sff{Sys}$ has two reductions to separate processes, as expected:
\begin{mathpar}
    \sff{Sys} \reddLazy \pRes{s} ( \sff{Server}_s^{\sff{buy}} \| \sff{Alice}_s^{\sff{buy}} )
    \and
    \text{and}
    \and
    \sff{Sys} \reddLazy \pRes{s} ( \sff{Server}_s^{\sff{peek}} \| \sff{Alice}_s^{\sff{peek}} )
\end{mathpar}
In contrast, the confluent non-deterministic choice from \cite{conf/esop/CairesP17}, denoted~`$\oplus$', behaves differently: in the confluent semantics, $\sff{Sys}$ reduces to a \emph{single} process including \emph{both alternatives}, i.e., $\sff{Sys} \redd \pRes{s} ( \sff{Server}_s^{\sff{buy}} \| \sff{Alice}_s^{\sff{buy}} ) \oplus \pRes{s} ( \sff{Server}_s^{\sff{peek}} \| \sff{Alice}_s^{\sff{peek}} )$.

\paragraph{Contributions.}

We study new concurrent \ifthesis\else and functional \fi calculi with usual (\nconf) forms of \nond.
Framed in the typed (resource-conscious) setting, we strive for definitions that do not exert a too drastic discarding of branches (as in the \nconf case) but also that do not exhibit a lack of commitment (as in the \conf case).
Concretely, we present\ifthesis~\else:

\smallskip \noindent
(\secref{s:pi})~\fi
\clpi, a variant of the session-typed \picalc in~\cite{conf/esop/CairesP17}, now with \nconf \nondt choice.
Its \emph{lazy} semantics adapts to the typed setting the usual semantics of non-deterministic choice in the untyped \picalc~\cite{book/SangiorgiW03}.
Well-typed processes enjoy type preservation and deadlock-freedom (\Cref{t:clpi:tp,t:clpi:df}).

\ifthesis\else
\smallskip \noindent
(\secref{s:lambda})
\lamcoldetshlin , a resource \lamcalc with   \nond, enhanced with constructs for expressing resource usage and failure.
Its  non-idempotent intersection type system provides a quantitative measure of the need/usage of resources.
Well-typed terms enjoy subject reduction and subject expansion (\Cref{t:lamSRShort,t:lamSEShort}).

\smallskip \noindent
(\secref{s:trans})
A typed translation of~\lamcoldetshlin into \clpi, which
provides further validation for our \nondt calculi, and casts them in the context of `{functions-as-processes}'.
We prove that our translation is \emph{correct}, i.e., it preserves types and satisfies tight operational correspondences (\Cref{def:encod_judge,t:correncLazy}).

\smallskip \noindent
Moreover, \secref{s:disc} closes by discussing related works.
\ifaplas
Omitted technical material
can be found in the full version of the paper~\cite{DBLP:journals/corr/abs-2205-00680}.
\else
Due to space limitations, appendices contain (i)~omitted material (in particular, proofs of technical results); (ii)~an alternative \emph{eager} semantics for \clpi, which we  compare against the lazy semantics; and (iii)~extensions of \clpi and \lamcoldetshlin with  \emph{unrestricted} resources.
\fi
\fi

\ifthesis
\Cref{s:clpi:clpi} introduces \clpi and its meta-theoretical properties.
\Cref{s:clpi:eager} discusses an alternative, eager semantics for \clpi that is closer the traditional semantics of the \picalc, and compares it to the lazy semantics presented in \Cref{s:clpi:clpi}.
\Cref{s:clpi:rw} discusses related work, and \Cref{s:clpi:concl} draws conclusions.
\fi

\section{A Typed \texorpdfstring{$\pi$}{Pi}-calculus with Non-deterministic Choice}
\label{s:pi}

We introduce \clpi, a session-typed \picalc with \nondt choice.
Following~\cite{conf/esop/CairesP17}, session types express protocols to be executed along channels.
These protocols can be \emph{non-deterministic}: sessions may succeed but also fail.
The novelty in \clpi is the \nondt choice operator `$P \nd Q$', whose \emph{lazily committing semantics} is compatible with linearity.
Key properties of well-typed processes are \emph{type preservation} and \emph{deadlock-freedom}.

\subsection{Syntax and Semantics}

\begin{figure}[t]
    \begin{align*}
        P, Q &
        \begin{array}[t]{@{}l@{}lr@{\kern2em}l@{\kern1em}lr@{}}
            {} ::= {} &
            \0 & \text{inaction}
            & \sepr &
            \pFwd [x<>y] & \text{forwarder}
            \\ \sepr* &
            \pRes{x} ( P \| Q ) & \text{connect}
            & \sepr &
            P \nd Q & \text{non-determinism}
            \\ \sepr* &
            \pOut* x[y] ; ( P \| Q ) & \text{send}
            & \sepr &
            \pIn x(y) ; P & \text{receive}
            \\ \sepr* &
            \pSel* x < \ell ; P & \text{select}
            & \sepr &
            \pBra* x > \{ i : P \}_{i \in I} & \text{branch}
            \\ \sepr* &
            \pCli* x[y];P & \text{client request}
            & \sepr &
            \pSrv x(y) ; P & \text{server}
            \\ \sepr* &
            \pClose x[] & \text{close}
            & \sepr &
            \pWait x() ; P & \text{wait}
            \\ \sepr* &
            \pSomeIn x{w_1,\ldots,w_n} ; P & \text{expect}
            & \sepr &
            \pSome x ; P & \text{available}
            \\ \sepr* &
            P \| Q & \text{parallel}
            & \sepr &
            \pNone x & \text{unavailable}
        \end{array}
    \end{align*}

    \dashes
    \vspace{-2em}
    \begin{mathpar}
        \begin{bussproof}[cong-alpha]
            \bussAssume{
                P \equiv_\alpha Q
            }
            \bussUn{
                P \equiv Q
            }
        \end{bussproof}
        \and
        \begin{bussproof}[cong-fwd-symm]
            \bussAx{
                \pFwd [x<>y] \equiv \pFwd [y<>x]
            }
        \end{bussproof}
        \and
        \begin{bussproof}[cong-par-unit]
            \bussAx{
                P \| \0 \equiv P
            }
        \end{bussproof}
        \and
        \begin{bussproof}[cong-par-assoc]
            \bussAx{
                P \| (Q \| R) \equiv (P \| Q) \| R
            }
        \end{bussproof}
        \and
        \begin{bussproof}[cong-par-comm]
            \bussAx{
                P \| Q \equiv Q \| P
            }
        \end{bussproof}
        \and
        \begin{bussproof}[cong-conn-symm]
            \bussAx{
                \pRes{x} ( P \| Q ) \equiv \pRes{x} ( Q \| P )
            }
        \end{bussproof}
        \and
        \begin{bussproof}[cong-conn-par]
            \bussAssume{
                x \notin \fn(Q)
            }
            \bussUn{
                \pRes{x} \big( ( P \| Q ) \| R \big) \equiv \pRes{x} ( P \| R ) \| Q
            }
        \end{bussproof}
        \and
        \begin{bussproof}[cong-conn-conn]
            \bussAssume{
                x \notin \fn(Q)
            }
            \bussAssume{
                y \notin \fn(R)
            }
            \bussBin{
                \pRes{x} \big( \pRes{y} ( P \| Q ) \| R \big) \equiv \pRes{y} \big( \pRes{x} ( P \| R ) \| Q \big)
            }
        \end{bussproof}
        \and
        \begin{bussproof}[cong-nd-refl]
            \bussAx{
                P \nd P \equiv P
            }
        \end{bussproof}
        \and
        \begin{bussproof}[cong-nd-symm]
            \bussAx{
                P \nd Q \equiv Q \nd P
            }
        \end{bussproof}
        \and
        \begin{bussproof}[cong-nd-assoc]
            \bussAx{
                ( P \nd Q ) \nd R \equiv P \nd ( Q \nd R )
            }
        \end{bussproof}
        \and
        \begin{bussproof}[cong-srv]
            \bussAssume{
                x \notin \fn(Q)
            }
            \bussUn{
                \pRes{x} ( \pSrv x(y) ; P \| Q ) \equiv Q
            }
        \end{bussproof}
    \end{mathpar}

    \caption{Syntax (top) and structural congruence (bottom) for \clpi.}
    \label{f:clpi:lang}
\end{figure}

We use $P,Q,\ldots$ to denote processes, and $x,y,z,\ldots$ to denote \emph{names} representing  channels.
\Cref{f:clpi:lang} (top) gives the syntax of processes.
$P \{ y/z \}$ denotes the capture-avoiding substitution of $y$ for $z$ in $P$.
Process~$\0$ denotes inaction, and $\pFwd [x<>y]$ is a forwarder: a bidirectional link between $x$ and $y$.
Parallel composition appears in two forms:
while the process $P \| Q$ denotes communication-free concurrency,
process $\pRes{x} ( P \| Q )$ uses restriction to express that $P$ and $Q$ implement complementary (session) behaviors on $x$ and do not share any other names.

Process $P \nd Q$ denotes the non-deterministic choice between $P$ and $Q$: intuitively, if one choice can perform a synchronization, the other option may be discarded if it cannot.
Since~$\nd$ is associative, we often omit parentheses.
Also, we write $\bignd_{i \in I} P_i$ for the non-deterministic choice between each $P_i$ for $i \in I$.

Our output construct integrates parallel composition and restriction: process $\pOut* x[y] ; ( P \| Q )$ sends a fresh name $y$ along $x$ and then continues as $P \| Q$.
The type system will ensure that behaviors on $y$ and $x$ are implemented by $P$ and $Q$, respectively, which do not share any names---this separation defines communication-free concurrency and is key to ensuring deadlock-freedom.
The input process $\pIn x(y) ; P$ receives a name $z$ along $x$ and continues as $P \{ z/y \}$, which does not require the separation present in the output case (as $x$ and $y$ are implemented separately as per the corresponding output).
Process $\pBra* x > \{ i : P_i \}_{i \in I}$ denotes a branch with labeled choices indexed by the finite set $I$: it awaits a choice on $x$ with continuation $P_j$ for each $j \in I$.
The process $\pSel* x < \ell ; P$ selects on $x$ the choice labeled $\ell$ before continuing as~$P$.
The server process $\pSrv x(y) ; P$ accepts requests from clients, receiving a name $z$ along $x$ to spawn $P \{ z/y \}$; the server process remains available for further client requests.
A client request $\pCli* x[y] ; P$ sends a fresh name $y$ along $x$ and continues as $P$.
Processes $\pClose x[]$ and $\pWait x() ; P$ are dual actions for closing the session on $x$.
\ifthesis\else
We omit replicated servers $\pSrv x(y);P$ and corresponding client requests $\pCli* x[y];P$, but they can be easily added 
\ifaplas
(cf.\ \cite{DBLP:journals/corr/abs-2205-00680}).
\else
\ifappendix (cf.~\secref{as:fullPi})\else (cf.~\cite{thesis/vdHeuvel24ex})\fi.
\fi
\fi

The remaining constructs define non-deterministic sessions which may provide a protocol or fail, following~\cite{conf/esop/CairesP17}.
Process $\pSome x ; P$ confirms the availability of a session on $x$ and continues as $P$.
Process $\pNone x$ signals the failure to provide the session on $x$.
Process $\pSomeIn x{w_1,\ldots,w_n} ; P$ specifies a dependency on a non-deterministic session on $x$.
This process can either (i)~synchronize with an action $\pSome x$ and continue as $P$, or (ii)~synchronize with an action $\pNone x$, discard $P$, and propagate the failure on $x$ to $w_1,\ldots,w_n$ (which implement sessions in $P$).

Name $y$ is bound in $\pRes{y} ( P \| Q )$, $\pOut* x[y] ; ( P \| Q )$, and $\pIn x(y) ; P$.
\ifaplas
The set $\fn(P)$ includes the names in $P$ that are not bound.
\else
We write $\fn(P)$ and $\bn(P)$ for the free and bound names of $P$, respectively, and $\fln(P)$ and $\fpn(P)$ for the free linear/non-linear (for client requests and servers) names of $P$, respectively; note that $\fpn(P) = \fn(P) \setminus \fln(P)$.
\fi
We adopt Barendregt's convention.
To reduce eye strain, in writing $\pSomeIn x{}$ we freely combine names and sets of names.
This way, e.g., we write $\pSomeIn x{y,\fn(P),\fn(Q)}$ rather than $\pSomeIn x{\{y\} \cup \fn(P) \cup \fn(Q)}$.

\paragraph{Structural congruence.}

Reduction defines the steps that a process performs on its own.
It relies on \emph{structural congruence} ($\equiv$), the least congruence relation on processes induced by the rules in \Cref{f:clpi:lang} (bottom).
Like the syntax of processes, the definition of $\equiv$ is aligned with the type system (defined next), such that $\equiv$ preserves typing (subject congruence, cf.\ \Cref{t:clpi:tp}).
\ifthesis
In the output $\pOut* x[y] ; ( P \| Q )$ the type system will ensure that $P$ provides behavior for $y$ and $Q$ for $x$.
Had we defined the syntax of output simply as $\pOut* x[y] ; P'$, we would have that $\pOut* x[y] ; ( P \| Q ) \equiv \pOut* x[y] ; ( Q \| P )$: on the right-hand-side, the behaviors of $x$ and $y$ are provided on the wrong sides of the parallel composition, making the process ill-typed.
By baking the parallel composition into the syntax of output, the expected symmetry of parallel composition does not apply under the output prefix, but only inside the parallel components of an output.
Note that $\pRes{x} ( P \| Q )$ also bakes in a parallel composition, but symmetry under restriction should hold; we therefore include an explicit rule for it.
However, associativity of parallel composition should and does not hold directly under restriction.
We also have a rule for scope extrusion of independent parallel compositions and restrictions, and a rule for removing unused servers.
\fi
Differently from~\cite{conf/esop/CairesP17}, we do not allow distributing non-deterministic choice over parallel and restriction.
As shown in \ifaplas \cite{DBLP:journals/corr/abs-2205-00680}\else \Cref{s:clpi:eager}\fi, the position of a non-deterministic choice in a process determines how it may commit, so changing its position affects commitment.

\paragraph{Reduction: intuitions and prerequisites.}

Barring non-deterministic choice, our reduction rules arise as directed interpretations of proof transformations in the underlying linear logic.
We follow Caires and Pfenning~\cite{conf/concur/CairesP10} and Wadler~\cite{conf/icfp/Wadler12} in interpreting cut-elimination in linear logic as synchronization in \clpi.

Before delving into our reduction rules (\Cref{f:clpi:reddLazy}), it may be helpful to consider the usual reduction axiom for the (untyped) \picalc (e.g.,~\cite{journal/ic/MilnerPW92,book/SangiorgiW03}):
\begin{align}
    ( \pOut* x[z] ; P_1 + M_1 ) \| ( \pIn x(y) ; P_2 + M_2 )
    \redd
    P_1 \| P_2 \{ z / y \}
    \label{eq:clpi:usualpi}
\end{align}
This axiom captures the interaction of two (binary) choices: it integrates the commitment of choice in synchronization; after the reduction step, the two branches not involved in the synchronization, $M_1$ and $M_2$, are discarded.
Our semantics of \clpi is defined similarly: when a prefix within a branch of a choice synchronizes with its dual, that branch reduces and the entire process commits to it.

The key question at this point is: when and to which branches should we commit?
In~\eqref{eq:clpi:usualpi}, a communication commits to a single branch.
Here, we define a \emph{lazy semantics} that minimizes commitment as much as possible.

The intuitive idea is that multiple branches of a choice may contain the same prefix, and so all these branches represent possibilities for synchronization (``possible branches'').
Other branches with different prefixes denote different possibilities (``impossible branches'').
When one synchronization is chosen, the possible branches are maintained while the impossible ones are discarded.

\begin{example}
    \label{x:clpi:possible}
    To distinguish possible and impossible branches, consider:
    \begin{align*}
        P \deq \pRes{s} \big( \pBra* s > \{ \sff{buy} : \ldots , \sff{peek} : \ldots \} \| ( \pSel* s < \sff{buy} ; \ldots \nd \pSel* s < \sff{buy} ; \ldots \nd \pSel* s < \sff{peek} ; \ldots ) \big)
    \end{align*}
    The branch construct (case) provides the context for the non-deterministic choice.
    When the case synchronizes on the `\sff{buy}' label, the two branches prefixed by `$\pSel* s < \sff{buy}$' are possible, whereas the branch prefixed by `$\pSel* s < \sff{peek}$' becomes impossible, and can be discarded.
    The converse occurs when the `\sff{peek}' label is selected.
\end{example}

To formalize these intuitions, our lazy semantics (\Cref{f:clpi:reddLazy})  relies on some auxiliary definitions.
First, we define contexts.

\begin{definition}
    \label{d:clpi:ndCtx}
    We define \emph{ND-contexts} ($\evalCtx{N},\evalCtx{M}$) as follows:
    \[
        \evalCtx{N},\evalCtx{M} ::= \evalHole \sepr \evalCtx{N} \| P \sepr \pRes{x} ( \evalCtx{N} \| P ) \sepr \evalCtx{N} \nd P
    \]
    The process obtained by replacing $\evalHole$ in $\evalCtx{N}$ with $P$ is denoted $\evalCtx{N}[P]$.
    We refer to ND-contexts that do not use the clause `$\evalCtx{N} \nd P$' as \emph{D-contexts}, denoted $\evalCtx{C},\evalCtx{D}$.
\end{definition}

Using D-contexts, we can express that, e.g., $\bignd_{i \in I} \evalCtx{C_i}[\pClose x[]]$ and $\bignd_{j \in J} \evalCtx{D_j}[\pWait x() ; Q_j]$ should match.
To account for reductions with impossible branches, we define a precongruence on processes, denoted $\precongLazy{S}$, where the parameter $S$ denotes the subject(s) of the prefix in the possible branches.
Our semantics is then closed under $\precongLazy{S}$.
Hence, e.g., anticipating a reduction on $x$, the possible branch $\evalCtx{C_1}[\pIn x(y) ; P]$ can be extended with an impossible branch to form $\evalCtx{C_1}[\pIn x(y) ; P] \nd \evalCtx{C_2}[\pWait z() ; Q]$.

Before defining $\precongLazy{S}$ (\Cref{d:clpi:precongLazy}), we first define prefixes (and their subjects).
Below, we write $\widetilde{x}$ to denote a finite tuple of names $x_1,\ldots,x_k$.

\begin{definition}
    \label{d:clpi:prefix}
    Prefixes are defined as follows:
    \[
        \alpha, \beta ::=
        \pOut* x[y] \sepr \pIn x(y) \sepr \pSel* x < \ell \sepr \pBra* x > {} \sepr \pCli* x[y] \sepr \pSrv x(y)
        \sepr \pClose x[] \sepr \pWait x() \sepr
        \pSome x \sepr \pNone x \sepr \pSomeIn x{\widetilde{w}} \sepr \pFwd [x<>y]
    \]
    The subjects of $\alpha$, denoted $\subjs(\alpha)$, are $\{x\}$ or $\{x,y\}$ in case of $\pFwd [x<>y]$.
    By abuse of notation, we write $\alpha ; P$ even when $\alpha$ takes no continuation or multiple (as in, e.g., $\pFwd [x<>y]$ and $\pOut* x[y]$).
\end{definition}

\begin{definition}
    \label{d:clpi:precongLazy}
    Let $\prefRel$ denote the least relation on prefixes (\Cref{d:clpi:prefix}) defined by: \\
    (i)~$\pOut* x[y] \prefRel \pOut* x[z]$, (ii)~$\pIn x(y) \prefRel \pIn x(z)$, and (iii)~$\alpha \prefRel \alpha$ otherwise.

    \medskip
    Given non-empty $S \subseteq \{x,y\}$, the precongruence $P \precongLazy{S} Q$ holds when both conditions hold:
    \begin{enumerate}
        \item\label{i:clpi:precongLazySingleton}
            $S = \{x\}$
            implies
            \\
            $P = \Big( \bignd_{i \in I} \evalCtx{C_i}[\alpha_i ; P_i] \Big) \nd \Big( \bignd_{j \in J} \evalCtx{C_j}[\beta_j ; Q_j] \Big)$
            and
            $Q = \bignd_{i \in I} \evalCtx{C_i}[\alpha_i ; P_i]$,
            where
            \\
            (i)~$\forall i,i' \in I.~ \alpha_i \prefRel \alpha_{i'}$ and $\subjs(\alpha_i) = \{x\}$,
            and \\
            (ii)~$\forall i \in I.~ \forall j \in J.~ \alpha_i \nprefRel \beta_j \land x \in \fn(\beta_j ; Q_j)$;

        \item\label{i:clpi:precongLazyForwarder}
            $S = \{x,y\}$
            implies
            \\
            $P = \Big( \bignd_{i \in I} \evalCtx{C_i}[\pFwd [x<>y]] \Big) \nd  \Big( \bignd_{j \in J} \evalCtx{C_j}[\pFwd [x<>z_j]] \Big) \nd \Big( \bignd_{k \in K} \evalCtx{C_k}[\alpha_k;P_k] \Big)$
            \\
            and
            $Q = \bignd_{i \in I} \evalCtx{C_i}[\pFwd [x<>y]]$,
            where
            \\
            (i)~$\forall j \in J.~ z_j \not = y$,
            and
            (ii)~$\forall k \in K.~ x \in \fn(\alpha_k;P_k) \land \forall z.~ \alpha_k \nprefRel \pFwd [x<>z]$.
    \end{enumerate}
\end{definition}

\noindent
Intuitively, while the relation $\prefRel$ allows us to equate output/input prefixes with the same subject (but different object), the rest of \Cref{d:clpi:precongLazy} accounts for two kinds of reduction: (i)~due to synchronizations (Item~\ref{i:clpi:precongLazySingleton}), in which case $S$ is a singleton, and (ii)~due to forwarding (Item~\ref{i:clpi:precongLazyForwarder}), in which case $S$ contains the two names involved.
In both cases, $S$ is used to discard ``impossible'' branches.

\begin{example}
    \label{x:clpi:precongLazy}
    Recall process $P$ from \Cref{x:clpi:possible}.
    To derive a synchronization with the `\sff{buy}' alternative of the case, we can use $\precongLazy{S}$ to discard the `\sff{peek}' alternative:
    \begin{align*}
        \pSel* s < \sff{buy} ; \ldots \nd \pSel* s < \sff{buy} ; \ldots \nd \pSel* s < \sff{peek} ; \ldots ~\precongLazy{s}~ \pSel* s < \sff{buy} ; \ldots \nd \pSel* s < \sff{buy} ; \ldots
        \tag*{\qedhere}
    \end{align*}
\end{example}

\paragraph{Reduction rules.}

\begin{figure}[p]
    \begin{mathpar}
        \begin{bussproof}[red-lazy-fwd]
            \bussAx{
                \pRes{x} ( \bignd_{i \in I} \evalCtx[\big]{C_i}[\pFwd [x<>y]] \| Q )
                \reddLazy_{x,y}
                \bignd_{i \in I} \evalCtx{C_i}[Q \{ y/x \}]
            }
        \end{bussproof}
        \and
        \begin{bussproof}[red-lazy-send-recv]
            \bussAx{
                \begin{array}{@{}r@{}l@{}}
                    &
                    \pRes{x} ( \bignd_{i \in I} \evalCtx{C_i}[\pOut* x[y_i] ; ( P_i \| Q_i )] \| \bignd_{j \in J} \evalCtx{D_j}[\pIn x(z) ; R_j] )
                    \\
                    \reddLazy_x
                    &
                    \bignd_{i \in I} \evalCtx*{C_i}[\pRes{x} \big( Q_i \| \pRes{w} ( P_i \{ w/y_i \} \| \bignd_{j \in J} \evalCtx{D_j}[R_j \{ w/z \}] ) \big)]
                \end{array}
            }
        \end{bussproof}
        \and
        \begin{bussproof}[red-lazy-sel-bra]
            \bussAssume{
                k' \in K
            }
            \bussUn{
                \pRes{x} ( \bignd_{i \in I} \evalCtx{C_i}[\pSel* x < k' ; P_i] \| \bignd_{j \in J} \evalCtx{D_j}[\pBra* x > \{ k : Q_j^k \}_{k \in K}] )
                \reddLazy_x
                \pRes{x} ( \bignd_{i \in I} \evalCtx{C_i}[P_i] \| \bignd_{j \in J} \evalCtx{D_j}[Q_j^{k'}] )
            }
        \end{bussproof}
        \and
        \begin{bussproof}[red-lazy-cli-srv]
            \bussAx{
                \begin{array}{@{}r@{}l@{}}
                    &
                    \pRes{x} ( \bignd_{i \in I} \evalCtx{C_i}[\pCli* x[y_i] ; P_i] \| \bignd_{j \in J} \evalCtx{D_j}[\pSrv x(z) ; Q_j] )
                    \\
                    \reddLazy_x
                    &
                    \bignd_{j \in J} \evalCtx*{D_j}[\pRes{x} ( \bignd_{i \in I} \evalCtx{C_i}[P_i \{ w/z \}] \| Q_j \{ w/z \} ) \| \pSrv x(z) ; Q_j]
                \end{array}
            }
        \end{bussproof}
        \and
        \begin{bussproof}[red-lazy-close-wait]
            \bussAx{
                \pRes{x} ( \bignd_{i \in I} \evalCtx{C_i}[\pClose x[]] \| \bignd_{j \in J} \evalCtx{D_j}[\pWait x() ; Q_j] )
                \reddLazy_x
                \bignd_{i \in I} \evalCtx{C_i}[\0] \| \bignd_{j \in J} \evalCtx{D_j}[Q_j]
            }
        \end{bussproof}
        \and
        \begin{bussproof}[red-lazy-some]
            \bussAx{
                \begin{array}{@{}r@{}l@{}}
                    &
                    \pRes{x} ( \bignd_{i \in I} \evalCtx{C_i}[\pSome x ; P_i] \| \bignd_{j \in J} \evalCtx{D_j}[\pSomeIn x{w_1,\ldots,w_n} ; Q_j] )
                    \\
                    \reddLazy_x
                    &
                    \pRes{x} ( \bignd_{i \in I} \evalCtx{C_i}[P_i] \| \bignd_{j \in J} \evalCtx{D_j}[Q_j] )
                \end{array}
            }
        \end{bussproof}
        \and
        \begin{bussproof}[red-lazy-none]
            \bussAx{
                \begin{array}{@{}r@{}l@{}}
                    &
                    \pRes{x} ( \bignd_{i \in I} \evalCtx{C_i}[\pNone x] \| \bignd_{j \in J} \evalCtx{D_j}[\pSomeIn x{w_1,\ldots,w_n} ; Q_j] )
                    \\
                    \reddLazy_x
                    &
                    \pRes{x} ( \bignd_{i \in I} \evalCtx{C_i}[\0] \| \bignd_{j \in J} \evalCtx{D_j}[\pNone w_1 \| \ldots \| \pNone w_n] )
                \end{array}
            }
        \end{bussproof}
        \and
        \begin{bussproof}[red-lazy-precong]
            \bussAssume{
                x \in S
            }
            \bussAssume{
                P \precongLazy{S} P'
            }
            \bussAssume{
                Q \precongLazy{S} Q'
            }
            \bussAssume{
                \pRes{x} ( P' \| Q' ) \reddLazy_S R
            }
            \bussQuat{
                \pRes{x} ( P \| Q ) \reddLazy_S R
            }
        \end{bussproof}
        \and
        \begin{bussproof}[red-lazy-scope]
            \bussAssume{
                \pRes{x} ( P \| \evalCtx[\big]{N}[\evalCtx{C}[Q_1] \nd \evalCtx{C}[Q_2]] ) \reddLazy_S R
            }
            \bussUn{
                \pRes{x} ( P \| \evalCtx[\big]{N}[\evalCtx{C}[Q_1 \nd Q_2]] ) \reddLazy_S R
            }
        \end{bussproof}
        \and
        \begin{bussproof}[red-lazy-cong]
            \bussAssume{
                P \equiv P'
            }
            \bussAssume{
                P' \reddLazy_S Q'
            }
            \bussAssume{
                Q' \equiv Q
            }
            \bussTern{
                P \reddLazy_S Q
            }
        \end{bussproof}
        \and
        \begin{bussproof}[red-lazy-conn]
            \bussAssume{
                P \reddLazy_S P'
            }
            \bussUn{
                \pRes{x} ( P \| Q ) \reddLazy_S \pRes{x} ( P' \| Q )
            }
        \end{bussproof}
        \and
        \begin{bussproof}[red-lazy-par]
            \bussAssume{
                P \reddLazy_S P'
            }
            \bussUn{
                P \| Q \reddLazy_S P' \| Q
            }
        \end{bussproof}
        \and
        \begin{bussproof}[red-lazy-nd]
            \bussAssume{
                P \reddLazy_S P'
            }
            \bussUn{
                P \nd Q \reddLazy_S P' \nd Q
            }
        \end{bussproof}
    \end{mathpar}

    \caption{Lazy reduction semantics for \clpi.}
    \label{f:clpi:reddLazy}
\end{figure}

\Cref{f:clpi:reddLazy} gives the rules for the lazy semantics, denoted $\reddLazy_S$, where the set~$S$ contains the names involved in the interaction.
We omit the curly braces in this annotation; this way, e.g., we write `$\reddLazy_{x,y}$' instead of `$\reddLazy_{\{x,y\}}$'.
Also, we write $\reddLazy_S^k$ to denote a sequence of $k \geq 0$ reductions.

The first six rules in \Cref{f:clpi:reddLazy} formalize forwarding and communication: they are defined on choices containing different D-contexts (cf.\ \Cref{d:clpi:ndCtx}), each with the same prefix but possibly different continuations; these rules preserve the non-deterministic choices.
Rule~\ruleLabel{red-lazy-fwd} fixes $S$ to the forwarder's two names, and the other rules fix $S$ to the one involved name.
In particular, Rule~\ruleLabel{red-lazy-send-recv} formalizes name exchange: it involves multiple senders and multiple receivers (grouped in choices indexed by $I$ and $J$, respectively).
Because they proceed in lock-step, reduction leads to substitutions involving the same (fresh) name $w$; also, the scopes of the choice and the contexts enclosing the senders is extended.

Rule~\ruleLabel{red-lazy-precong} can be used to derive a synchronization that discards groups of choices.
Rule~\ruleLabel{red-lazy-scope} allows inferring reductions when non-deterministic choices are not top-level: e.g.,
\[
    \pRes{x} \big( \pClose x[] \| \pRes{y} ( ( \pWait x() ; Q_1 \nd \pWait x() ; Q_2 ) \| R ) \big) \reddLazy_x \pRes{y} ( Q_1 \| R ) \nd \pRes{y} ( Q_2 \| R ).
\]
The last four rules formalize that reduction is closed under structural congruence, restriction, parallel composition, and non-deterministic choice.

As mentioned earlier, a key motivation for our work is to have \nondt choices that effectively enforce commitment, without a too drastic discarding of alternatives.
Next we illustrate this intended form of \emph{gradual commitment}.

\begin{example}
    \label{x:clpi:payments}
    Consider the following variant of the movie server introduced in \Cref{s:clpi:intro}, where the handling of the payment is now modeled as a branch:
    \[
        \sff{NewServer}_s \deq \pIn s(\textit{title}) ; \pBra* s > \left\{
            \begin{array}{@{}l@{}}
                \sff{buy} : \pBra* s > \left\{
                    \begin{array}{@{}l@{}}
                        \sff{card} : \pIn s(\textit{info}) ; \pOut* s[\texttt{movie}] ; \pClose s[] ,
                        \\
                        \sff{cash} : \pOut* s[\texttt{movie}] ; \pClose s[]
                    \end{array}
                \right\}
                ,
                \\
                \sff{peek} : \pOut* s[\texttt{trailer}] ; \pClose s[]
            \end{array}
        \right\}
    \]
    Consider a client, Eve, who cannot decide between buying `Oppenheimer' or watching its trailer.
    In the former case, she has two options for payment method:
    \[
        \sff{Eve}_s \deq \pOut* s[\texttt{Oppenheimer}] ; \left(
            \begin{array}{@{}rl@{}}
                &
                \pSel* s < \sff{buy} ; \pSel* s < \sff{card} ; \pOut* s[\texttt{visa}] ; \pIn s(\textit{movie}) ; \pWait s() ; \0
                \\ {\nd} &
                \pSel* s < \sff{buy} ; \pSel* s < \sff{cash} ; \pIn s(\textit{movie}) ; \pWait s() ; \0
                \\ {\nd} &
                \pSel* s < \sff{peek} ; \pIn s(\textit{link}) ; \pWait s() ; \0
            \end{array}
        \right)
    \]

    Let $\sff{Sys}^* \deq \pRes{s} ( \sff{NewServer}_s \| \sff{Eve}_s )$.
    After sending the movie's title, Eve's choice (buying or watching the trailer) enables gradual commitment.
    We have:
    \begin{mathpar}
        \sff{Sys}^* \reddLazy_s^2 \pRes{s} \big( \pBra* s > \{ \sff{card} : \ldots , \sff{cash} : \ldots \} \| ( \pSel* s < \sff{card} ; \ldots \nd \pSel* s < \sff{cash} ; \ldots ) \big) =: \sff{Sys}^*_1
        \\
        \text{and}
        \and
        \sff{Sys}^* \reddLazy_s^2 \pRes{s} ( \pOut* s[\texttt{trailer}] ; \ldots \| \pIn s(\textit{trailer}) ; \ldots ) =: \sff{Sys}^*_2
    \end{mathpar}
    Process $\sff{Sys}^*_1$ represents the situation in which Eve has selected $\sff{buy}$, in which case the third alternative ($\pSel* s < \sff{peek}; \ldots$) can be discarded as an impossible branch.
    Process $\sff{Sys}^*_2$ represents the dual situation, in which Eve selects $\sff{pick}$ and the two alternatives prefixed by `$\pSel* s < \sff{buy}$' are impossible branches.
    From~$\sff{Sys}^*_1$, the selection of payment method completes the commitment to one alternative:
    \[
        \sff{Sys}^*_1 \reddLazy_s \pRes{s} ( \pIn s(\textit{info}) ; \ldots \| \pOut* s[\texttt{visa}] ; \ldots )
        \
        \text{and}
        \
        \sff{Sys}^*_1 \reddLazy_s \pRes{s} ( \pOut* s[\texttt{movie}] ; \ldots \| \pIn s(\textit{movie}) ; \ldots )
        \tag*{\qedhere}
    \]
\end{example}

\ifaplas
\else
In~\Cref{s:clpi:eager}, we discuss an alternative \emph{eager} semantics that commits to a single branch upon communication, as in~\eqref{eq:clpi:usualpi}.
\fi

\subsection{Resource Control for \texorpdfstring{\clpi}{spi+} via Session Types}
\label{ss:piTypeSys}

We define a session type system for \clpi, following `propositions-as-sessions'~\cite{conf/concur/CairesP10,conf/icfp/Wadler12}.
As already mentioned, in a session type system, resources are names that perform protocols:
the \emph{type assignment} $x:A$ says that $x$ should conform to the protocol specified by the session type $A$.
We give the syntax of types:
\[
    A,B ::= \1 \sepr \bot \sepr A \tensor B \sepr A \parr B \sepr \oplus \{ i : A \}_{i \in I} \sepr \& \{ i : A \}_{i \in I} \sepr {?}A \sepr {!}A \sepr \& A \sepr \oplus A
\]
The units $\1$ and $\bot$ type closed sessions.
We use $A \tensor B$ to type a name that first outputs a name of type~$A$ and then proceeds as specified by $B$.
Similarly, $A \parr B$ types a name that inputs a name of type $A$ before proceeding as specified by~$B$.
Types $\oplus \{ i : A_i \}_{i \in I}$ and $\& \{ i : A_i \}_{i \in I}$ are assigned to names that can select and offer a labeled choice, respectively.
Type ${?}A$ is assigned to a name that performs a request for a session of type $A$, and type ${!}A$ is assigned to a name that offers a server of type $A$: the modalities `${?}$' and `${!}$' define \emph{non-linear} (i.e., persistent) types.
Finally, $\& A$ is the type of a name that \emph{may produce} a behavior of type $A$, or fail; dually, $\oplus A$ types a name that \emph{may consume} a behavior of type $A$.

For any type $A$ we denote its \emph{dual} as $\ol{A}$.
Intuitively, we use dual types to avoid communication errors: the type at one end of a channel is the dual of the type at the opposite end.
Duality is an involution, defined as follows:
\begin{align*}
    \ol{\1} &\deq \bot
    & \ol{A \tensor B} &\deq \ol{A}\parr\ol{B}
    & \ol{\oplus \{ i : A_i \}_{i\in I}} &\deq \& \{ i : \ol{A_i} \}_{i\in I}
    & \ol{{?}A} &\deq {!}A
    & \ol{\& A} &\deq \oplus \ol{A}
    \\
    \ol{\bot} &\deq \1
    & \ol{A \parr B} &\deq \ol{A} \tensor \ol{B}
    & \ol{\& \{ i : A_i \}_{i\in I}} &\deq \oplus \{ i : \ol{A_i} \}_{i\in I}
    & \ol{{!}A} &\deq {?}A
    & \ol{\oplus A} &\deq \& \ol{A}
\end{align*}

Judgments are of the form $\vdash P \typInf \Gamma$, where $P$ is a process and $\Gamma$ is a context, a collection of type assignments.
In writing $\Gamma , x:A$, we assume $x \notin \dom(\Gamma)$.
We write $\dom(\Gamma)$ to denote the set of names appearing in $\Gamma$.
Letting $\lozenge \in \{?, \oplus\}$, we write $\lozenge\Gamma$ to denote that, for every $x:A \in \Gamma$, there exists $A'$ such that $A = \lozenge A'$.

\begin{figure}[!t]
    \begin{mathpar}
        \begin{bussproof}[typ-cut]
            \bussAssume{
                \vdash P \typInf \Gamma, x:A
            }
            \bussAssume{
                \vdash Q \typInf \Delta, x:\ol{A}
            }
            \bussBin{
                \vdash \pRes{x} ( P \| Q ) \typInf \Gamma, \Delta
            }
        \end{bussproof}
        \and
        \begin{bussproof}[typ-par]
            \bussAssume{
                \vdash P \typInf \Gamma
            }
            \bussAssume{
                \vdash Q \typInf \Delta
            }
            \bussBin{
                \vdash P \| Q \typInf \Gamma, \Delta
            }
        \end{bussproof}
        \and
        \begin{bussproof}[typ-nd]
            \bussAssume{
                \vdash P \typInf \Gamma
            }
            \bussAssume{
                \vdash Q \typInf \Gamma
            }
            \bussBin{
                \vdash P \nd Q \typInf \Gamma
            }
        \end{bussproof}
        \and
        \begin{bussproof}[typ-inact]
            \bussAx{
                \vdash \0 \typInf \emptyset
            }
        \end{bussproof}
        \and
        \begin{bussproof}[typ-fwd]
            \bussAx{
                \vdash \pFwd [x<>y] \typInf x:A, y:\ol{A}
            }
        \end{bussproof}
        \and
        \begin{bussproof}[typ-close]
            \bussAx{
                \vdash \pClose x[] \typInf x:\1
            }
        \end{bussproof}
        \and
        \begin{bussproof}[typ-wait]
            \bussAssume{
                \vdash P \typInf \Gamma
            }
            \bussUn{
                \vdash \pWait x() ; P \typInf \Gamma, x:\bot
            }
        \end{bussproof}
        \and
        \begin{bussproof}[typ-send]
            \bussAssume{
                \vdash P \typInf \Gamma, y:A
            }
            \bussAssume{
                \vdash Q \typInf \Delta, x:B
            }
            \bussBin{
                \vdash \pOut* x[y] ; ( P \| Q ) \typInf \Gamma, \Delta, x:A \tensor B
            }
        \end{bussproof}
        \and
        \begin{bussproof}[typ-recv]
            \bussAssume{
                \vdash P \typInf \Gamma, y:A, x:B
            }
            \bussUn{
                \vdash \pIn x(y) ; P \typInf \Gamma, x:A \parr B
            }
        \end{bussproof}
        \and
        \begin{bussproof}[typ-sel]
            \bussAssume{
                \vdash P \typInf \Gamma, x:A_j
            }
            \bussAssume{
                j \in I
            }
            \bussBin{
                \vdash \pSel* x < j ; P \typInf \Gamma, x:\oplus \{ i : A_i \}_{i \in I}
            }
        \end{bussproof}
        \and
        \begin{bussproof}[typ-bra]
            \bussAssume{
                \forall i \in I.~ \vdash P_i \typInf \Gamma, x:A_i
            }
            \bussUn{
                \vdash \pBra* x > \{ i : P_i \}_{i \in I} \typInf \Gamma, x:\& \{ i : A_i \}_{i \in I}
            }
        \end{bussproof}
        \and
        \begin{bussproof}[typ-some]
            \bussAssume{
                \vdash P \typInf \Gamma, x:A
            }
            \bussUn{
                \vdash \pSome x ; P \typInf \Gamma, x:\& A
            }
        \end{bussproof}
        \and
        \begin{bussproof}[typ-none]
            \bussAx{
                \vdash \pNone x \typInf x:\& A
            }
        \end{bussproof}
        \and
        \begin{bussproof}[typ-some-in]
            \bussAssume{
                \vdash P \typInf \& \Gamma, {?}\Delta, x:A
            }
            \bussUn{
                \vdash \pSomeIn x{\dom(\Gamma)} ; P \typInf \& \Gamma, {?}\Delta, x:\oplus A
            }
        \end{bussproof}
        \and
        \begin{bussproof}[typ-cli]
            \bussAssume{
                \vdash P \typInf \Gamma, y:A
            }
            \bussUn{
                \vdash \pCli* x[y] ; P \typInf \Gamma, x:{?}A
            }
        \end{bussproof}
        \and
        \begin{bussproof}[typ-srv]
            \bussAssume{
                \vdash P \typInf {?}\Gamma, y:A
            }
            \bussUn{
                \vdash \pSrv x(y) ; P \typInf {?}\Gamma, x:{!}A
            }
        \end{bussproof}
        \and
        \begin{bussproof}[typ-weaken]
            \bussAssume{
                \vdash P \typInf \Gamma
            }
            \bussUn{
                \vdash P \typInf \Gamma, x:{?}A
            }
        \end{bussproof}
        \and
        \begin{bussproof}[typ-contract]
            \bussAssume{
                \vdash P \typInf \Gamma, x:{?}A, x':{?}A
            }
            \bussUn{
                \vdash P \{ x/x' \} \typInf \Gamma, x:{?}A
            }
        \end{bussproof}
    \end{mathpar}

    \caption{Typing rules for \clpi.}
    \label{f:clpi:typing}
\end{figure}

\Cref{f:clpi:typing} gives the typing rules: they correspond to the rules in Curry-Howard interpretations of classical linear logic as session types (cf.\ Wadler~\cite{conf/icfp/Wadler12}), with the rules for~$\& A$ and $\oplus A$ extracted from~\cite{conf/esop/CairesP17}, and the additional Rule~\ruleLabel{typ-nd} for non-confluent non-deterministic choice, which modifies the confluent rule in~\cite{conf/esop/CairesP17}.

Most rules follow~\cite{conf/icfp/Wadler12}, so we focus on those related to non-determinism.
Rule~\ruleLabel{typ-some} types a process with a name whose behavior can be provided, while Rule~\ruleLabel{typ-none} types a name whose behavior cannot.
Rule~\ruleLabel{typ-some-in} types a process with a name $x$ whose behavior may not be available.
If the behavior is not available, all the sessions in the process must be canceled; hence, the rule requires all names to be typed under $\& A$.

Rule~\ruleLabel{typ-nd} types our new non-deterministic choice operator; the branches must be typable under the same typing context.
Hence, all branches denote the same sessions, which may be implemented differently.
In context of a synchronization, branches that are kept are able to synchronize, whereas the discarded branches are not; nonetheless, the remaining branches still represent different implementations of the same sessions.
Compared to the rule for non-determinism in~\cite{conf/esop/CairesP17}, we do not require processes to be typable under $\& A$.

\begin{example}
    Recall that, in \Cref{x:clpi:payments}, $\sff{Eve}_s$ has two options for payment method:
    \[
        \sff{Eve}_s \deq \pOut* s[\texttt{Oppenheimer}] ; \left(
            \begin{array}{@{}rl@{}}
                &
                \pSel* s < \sff{buy} ; \pSel* s < \sff{card} ; \pOut* s[\texttt{visa}] ; \pIn s(\textit{movie}) ; \pWait s() ; \0
                \\ {\nd} &
                \pSel* s < \sff{buy} ; \pSel* s < \sff{cash} ; \pIn s(\textit{movie}) ; \pWait s() ; \0
                \\ {\nd} &
                \pSel* s < \sff{peek} ; \pIn s(\textit{link}) ; \pWait s() ; \0
            \end{array}
        \right)
    \]
    However, the branches of the non-deterministic choice in $\sff{Eve}_s$ are three different implementations of the same session:
    all three sessions on $s$ are typable
    \[
        \oplus \big\{ \sff{buy} : \oplus \{ \sff{card} : \mathtt{C} \tensor \mathtt{M} \parr \bot , \sff{cash} : \mathtt{M} \parr \bot \} , \sff{peek} : \mathtt{L} \parr \bot \big\}
    \]
    (where $\mathtt{C},\mathtt{M},\mathtt{L}$ are primitive, self-dual data types).
\end{example}

\begin{example}[Unavailable Movies]
    \label{x:clpi:unavailableMovies}
    Consider the following variant of the movie server introduced in \Cref{s:clpi:intro}.
    This server lets a client buy a movie, but the movie may not be available.
    We model this scenario using non-deterministic choice and non-deterministically available sessions:
    \[
        \vdash \sff{BuyServ}_s \deq \pIn s(\textit{title}) ; ( \pNone s \nd \pSome s ; \pIn s(\textit{paym}) ; \pOut* s[\texttt{movie}] ; \pClose s[] ) \typInf s:\mathtt{T} \parr \big( \& ( \mathtt{P} \parr \mathtt{M} \tensor \1 ) \big),
    \]
    where $\mathtt{T},\mathtt{P},\mathtt{M}$ denote primitive, self-dual data-types.
    We now define a client Ada who buys a movie for Charles; Ada can only forward the movie to Charles if it is actually available:
    \begin{align*}
        \vdash \sff{Ada}_{s,u} &\deq \pOut* s[\texttt{Barbie}] ; \pSomeIn s{u} ; \pOut* s[\texttt{visa}] ; \pIn s(\textit{movie}) ; \pWait s() ; \pSome u ; \pOut* u[\textit{movie}] ; \pClose u[]
        \\
        &\hphantom{\deq } \typInf s:\mathtt{T} \tensor \big( \oplus ( \mathtt{P} \tensor \mathtt{M} \parr \bot ) \big) , u:\& ( \mathtt{M} \tensor \1 )
        \\
        \vdash \sff{Charles}_u &\deq \pSomeIn u{} ; \pIn u(\textit{movie}) ; \pWait u() ; \0 \typInf u:\oplus ( \mathtt{M} \parr \1 )
    \end{align*}
    Let $\sff{BuySys} \deq \pRes{s} \big( \sff{BuyServ}_s \| \pRes{u} ( \sff{Ada}_{s,u} \| \sff{Charles}_u ) \big)$.
    Depending on whether the server has the movie ``Barbie'' available, we have the following reductions:
    \[
        \sff{BuySys} \reddLazy_s^2 \pRes{u} ( \pNone u \| \sff{Charles}_u )
        ~~
        \text{or}
        ~~
        \sff{BuySys} \reddLazy_s^5 \pRes{u} ( \pSome u ; \ldots \| \sff{Charles}_u )
        \tag*{\qedhere}
    \]
\end{example}

\begin{example}[An Unreliable Service]
    \label{x:clpi:safeOrRisk}
    Imagine a service $\sff{S}$ that makes two selections but may fail before each of them.
    It can be assigned the session type:
    \[
        A \deq \& \oplus \big\{ l : \& \oplus \{ l : A_1 , l' : A_2 \} , l' : \& \oplus \{ l : A_3 , l' : A_4 \} \big\}
    \]
    where $A_1, \ldots, A_4$ denote different  protocols.

    Consider a process $\sff{C}$ such that $\vdash \sff{C} \typInf x:\ol{A}, w:A$, i.e., $\sff{C}$ provides a session of type $A$ on name $w$.
    A possible implementation for $\sff{C}$ follows:
    it is defined as a non-deterministic choice between two different possibilities: (a)~relying on $\sff{S}$ to provide~$w$, risking a full-on crash, or (b)~providing $w$ without relying on~$\sff{S}$.
    \begin{align*}
        \sff{C} &\deq \begin{array}[t]{@{}l@{}}
            \pSomeIn x{w} ; \pSome w ; \pBra* x > \{ l : \pSel* w < l ; \sff{C}_1 , l' : \pSel* w < l' ; \sff{C}_1 \}
            \\ {} \nd
            \big( \pSome w ; ( \pSel* w < l ; \sff{D}_1 \nd \pSel* w < l' ; \sff{D}_1 ) \| \pSomeIn x{} ; \pBra* x > \{ l : \sff{D}_2 , l' : \sff{D}_2 \} \big)
        \end{array}
        \\
        \sff{C}_1 &\deq \pSomeIn x{w} ; \pSome w ; \pBra* x > \{ l : \pSel* w < l ; \pWait x() ; \pClose w[] , l' : \pSel* w < l' ; \pWait x() ; \pClose w[] \}
        \\
        \sff{D}_1 &\deq \pSome w ; ( \pSel* w < l ; \pClose w[] \nd \pSel* w < l' ; \pClose w[] )
        \qquad
        \sff{D}_2 \deq \pSomeIn x{} ; \pBra* x > \{ l : \pWait x() ; \0 , l' : \pWait x() ; \0 \}
    \end{align*}
    where, for simplicity, we assume $A_i = \1$ ($i \in \{1,\ldots,4\}$).

    In process $\sff{C}$, the ``risky'' alternative appears in the left branch, and relies on process $\sff{C}_1$: sessions on $x$ and $w$ occur intertwined, with each action on $w$ performing the dual behavior of $x$; e.g., a branching on $x$ is followed by a selection on $w$.
    The ``safe'' alternative appears in the right branch:  this is a very different design, as the behavior along $w$ and $x$ is separately implemented by $\sff{D}_1$ and $\sff{D}_2$, respectively.
    In this case, since $\sff{C}$ gets no input from $\sff{S}$ in this case, $\sff{D}_2$ makes its selections non-deterministically.
    Thus, despite their different implementations, risky and safe branches have the same types (i.e., realize the same sessions).
\end{example}

As expected, our type system does not ensure confluence,  but it does satisfy \emph{session fidelity} and \emph{communication safety}.
These are essential correctness properties:
the former states that processes correctly follow their ascribed session protocols, and the latter that no communication errors and mismatches occur.
Both properties follow directly from the fact that typing is consistent across structural congruence and reduction.
\ifaplas
See~\cite{DBLP:journals/corr/abs-2205-00680} for details.
\else
See~\ifappendix\Cref{as:clpi:tp}\else\cite{thesis/vdHeuvel24ex}\fi for details.
\fi

\begin{theorem}[Type Preservation]
    \label{t:clpi:tp}
    If $\vdash P \typInf \Gamma$, then both $P \equiv Q$ and $P \reddLazy_S Q$ (for any $Q$ and $S$) imply $\vdash Q \typInf \Gamma$.
\end{theorem}

\begin{sketch}
    By induction on the derivation of $P \equiv Q$ or $P \reddLazy_S Q$.
    In each case, we apply induction on the derivation of $\vdash P \typInf \Gamma$.
    Then, by inversion we discover the exact typing of $P$ and derive the typing of $Q$.
\end{sketch}

Another important, if often elusive, property in session types is \emph{dead\-lock-free\-dom}, which ensures that processes can reduce as long as they are not inactive.
Our type system satisfies deadlock-freedom for processes with fully connected names, i.e., typable under the empty context.
\ifaplas
See~\cite{DBLP:journals/corr/abs-2205-00680} for details.
\else
See \ifappendix\Cref{as:clpi:df}\else\cite{thesis/vdHeuvel24ex}\fi for details.
\fi

\begin{theorem}[Deadlock-freedom]
    \label{t:clpi:df}
    If $\vdash P \typInf \emptyset$ and $P \not\equiv \0$, then there are $Q$ and $S$ such that $P \reddLazy_S Q$.
\end{theorem}

\begin{sketch}
    In \Cref{s:clpi:eager}, we define an \emph{eager} semantics for \clpi.
    For this semantics, we prove deadlock-freedom in a fairly standard way (such as in~\cite{conf/esop/CairesP17}).
    We transfer deadlock-freedom from the eager semantics to the lazy semantics by showing that a reduction under the eager semantics implies a reduction under the lazy semantics.
\end{sketch}

\section{A Non-deterministic Resource \lamcalc}\label{s:lambda}

    We present \lamcoldetshlin, a resource \lamcalc
with   \nond  and lazy evaluation.
In \lamcoldetshlin,   \nond is  \col and \emph{implicit}, as it arises from the fetching of terms from {bags} of \emph{linear} resources.
This is different from  \clpi, where the choice operator `$\nd$'
specifies \nond \emph{explicitly}.
A mismatch between the number of variable occurrences and the size of the bag induces \emph{failure}.

In $\lamcoldetshlin$, the \emph{sharing} construct  $M\sharing{x_1,\ldots, x_n}{x}$,
expresses that
$x$  may be used in $M$ under ``aliases'' $x_1,\ldots, x_n$.
Hence, it atomizes $n$ occurrences of   $x$   in~$M$, via an explicit pointer to $n$ variables.
This way, e.g.,
the $\lambda$-term $\lambda x.(x\ x)$ is expressed in $\lamcoldetshlin$ as $\lambda x. (x_1 \bag{x_2}\sharing{x_1,x_2}{x})$, where $\bag{x_2}$ is a bag containing~$x_2$.

\begin{figure}[t]
    \begin{mdframed} \mysmall
        \begin{align*}
            M,N,L ::=~
            & x & \text{variable}
            & \quad
            {}\sepr M \esubst{C}{x} & \text{intermediate subst.}
            \\
            \sepr~
            & (M\ C) & \text{application}
            & \quad
            {}\sepr M \linexsub{C/\widetilde{x}} & \text{explicit subst.}
            \\
            \sepr~
            & \lambda x.M & \text{abstraction}
            & \quad
            {}\sepr \fail^{\tilde{x}} & \text{failure}
            \\
            \sepr~
            & M\sharing{\widetilde{x}}{x} & \text{sharing}
            \\
            C,D ::=~
            & \oneb \sepr \bag{M} \cdot\, C
            \span\span
            & \text{ bag}
            \\
            \lctx{C} ::=~
            & \hole \sepr \lctx{C}\sharing{\widetilde{x}}{x} \sepr (\lctx{C}\ C) \sepr \lctx{C} \linexsub{C/\widetilde{x}}
            \span\span
            & \text{context}
        \end{align*}
    \end{mdframed}
    \caption{Syntax of $\lamcoldetshlin$.}
    \label{f:lambdalin}
\end{figure}

\subsection{Syntax and Reduction Semantics}

\paragraph*{Syntax.}
We use $x,y,z,\ldots $ for variables, and write $\widetilde{x}$ to denote a finite sequence of pairwise distinct $x_i$'s, with length $|\widetilde{x}|$.
\Cref{f:lambdalin} gives the syntax of terms ($M,N,L$) and bags~($C, D$).
The empty bag is denoted $\oneb$.
We use $C_i$ to denote the $i$-th term in $C$, and $\size{C}$ denotes the number of elements in $C$.
To ease readability, we often write, e.g., $\bag{N_1, N_2}$
as a shorthand notation for
$\bag{N_1} \cdot \bag{N_2}$.

In  $M\sharing{\widetilde{x}}{x}$, we say that $\widetilde{x}$ are the \emph{shared variables}  and that $x$  is the \emph{sharing variable}.
We require for each $x_i \in \widetilde{x}$: (i)~$x_i$ occurs exactly once in $M$; (ii)~$x_i$ is not a sharing variable.
The sequence $\widetilde{x}$ can be empty:
$M\sharing{}{x}$ means that $x$ does not share any variables in $M$.
Sharing binds the shared variables in the term.

An abstraction $\lambda x. M$ binds
occurrences  of  $x$ in $M$.
Application $(M\ C)$ is as usual.
The term  $M \linexsub{C /  \widetilde{x}}$  is the \emph{explicit substitution}   of a bag $C$ for $\widetilde{x}$ in  $M$.
We require $\size{C} = |\widetilde{x}|$ and for each $x_i \in \widetilde{x}$: (i)~$x_i$ occurs in $M$; (ii)~$x_i$ is not a sharing variable; (iii)~$x_i$ cannot occur in another explicit substitution in $M$.
The term $M\esubst{ C }{ x }$ denotes an intermediate explicit substitution that does not (necessarily) satisfy the conditions for explicit substitutions.

The term $\fail^{\tilde{x}}$ denotes failure; the variables in $\widetilde{x}$ are ``dangling'' resources, which cannot be accounted for after failure.
We write $\lfv{M}$ to denote the free variables of $M$, defined as expected.
Term $M$ is \emph{closed} if $\lfv{M} = \emptyset$.

\begin{figure}[t]
    \begin{mdframed} \mysmall
        \begin{mathpar}
            \inferrule[$\redlab{RS{:}Beta}$]{ }{
                (\lambda x . M) C  \redd M\esubst{ C }{ x }
            }
            \and
            \inferrule[$\redlab{RS{:}Ex \dash Sub}$]{
                \size{C} = |\widetilde{x}|
                \\
                M \not= \fail^{\tilde{y}}
            }{
                (M\sharing{\widetilde{x}}{x})\esubst{ C }{ x } \redd  M\linexsub{C  /  \widetilde{x}}
            }
            \and
                        \inferrule[$\redlab{RS:TCont}$]{
                M \redd    N
            }{
                \lctx{C}[M] \redd   \lctx{C}[N]
            }
            \and
            \inferrule[$\redlab{RS{:}Fetch^{\ell}}$]{
                \headf{M} =  {x}_j
                \\
                0 < i \leq \size{C}
            }{
                M \linexsub{C /  \widetilde{x}, x_j} \redd  M \headlin{ C_i / x_j }  \linexsub{(C \setminus C_i ) /  \widetilde{x}  }
            }
            \and
            \inferrule[$\redlab{RS{:}Fail^{\ell}}$]{
                \size{C} \neq |\widetilde{x}|
                \\
                \widetilde{y} = (\lfv{M} \setminus \{  \widetilde{x}\} ) \cup \lfv{C}
            }{
                (M\sharing{\widetilde{x}}{x}) \esubst{C }{ x }  \redd  \fail^{\tilde{y}}
            }
            \and
            \inferrule[$\redlab{RS{:}Cons_1}$]{
                \widetilde{y} = \lfv{C}
            }{
                \fail^{\tilde{x}}\ C  \redd  \fail^{\tilde{x} \cup \widetilde{y}}
            }
            \and
            \inferrule[$\redlab{RS{:}Cons_2}$]{
                \size{C} =   |  {\widetilde{x}} |
                \\
                \widetilde{z} = \lfv{C}
            }{
                (\fail^{ {\tilde{x}} \cup \tilde{y}} \sharing{\widetilde{x}}{x})\esubst{ C }{ x }  \redd  \fail^{\tilde{y} \cup \widetilde{z}}
            }
            \and
            \inferrule[$\redlab{RS{:}Cons_3}$]{
                \widetilde{z} = \lfv{C}
            }{
                \fail^{\tilde{y}\cup \tilde{x}} \linexsub{C /  \widetilde{x}} \redd  \fail^{\tilde{y} \cup \widetilde{z}}
            }
        \end{mathpar}
        where $\headf{M}$ is defined as follows:
        \begin{align*}
            \headf{ x } &= x
            &
            \headf{ \lambda x.M } &= \lambda x.M
            &
            \headf{ (M\ C) } &= \headf{M}
            \\
            \headf{ \fail^{\widetilde{x}} } &= \fail^{\widetilde{x}}
            &
            \headf{ M \esubst{ C }{ x } } &= M \esubst{ C }{ x }
            &
            \headf{ M \linexsub{C / \widetilde{x}} } &= \headf{M}
            \\
            \headf{ M\sharing{\widetilde{x}}{x} } = \begin{cases}
                x & \text{$\headf{ M } = y$ and $y \in  {\widetilde{x}}$}
                \\
                \headf{ M } & \text{otherwise}
            \end{cases}
            \span\span\span\span
        \end{align*}
    \end{mdframed}
    \caption{Reduction rules for $\lamcoldetshlin$.}
    \label{fig:reduc_intermlin}\label{f:lambda_redlin}
\end{figure}

\paragraph*{Semantics.}
\Cref{fig:reduc_intermlin} gives the reduction semantics, denoted $\red$, and  the \emph{head variable} of term $M$, denoted $\headf{M}$.
 Rule~$\redlab{RS:Beta}$ induces an intermediate substitution.
Rule~$\redlab{RS:Ex{\dash}Sub}$  reduces an intermediate substitution to an explicit substitution, provided the size of the bag equals the number of shared variables.
In case of a mismatch, the term evolves into failure via Rule~$\redlab{RS:Fail^\ell}$.

An explicit substitution $M \linexsub{C/\widetilde{x}}$, where the head variable of $M$ is $x_j \in \widetilde{x}$, reduces via Rule~$\redlab{R:Fetch^\ell}$.
The rule removes a $C_i$ from $C$ (for some $0 < i \leq \size{C}$) and substitutes it for $x_j$ in $M$; this is how fetching induces a non-deterministic choice between $\size{C}$ possible reductions.
Rules~$\redlab{RS:Cons_j}$ for $j \in \{1,2,3\}$ consume terms when they meet failure.
Finally, Rule~$\redlab{RS:TCont}$ closes reduction under contexts.
The following example illustrates reduction.

\begin{example}\label{ex:syntax}\label{ex:lambdaRed}
    Consider the term $M_0 = ( \lambda x. x_1 \bag{x_2 \bag{x_3\ \oneb} } \sharing{ \widetilde{x} }{x} )\ \bag{\fail^{\emptyset} , y ,I\,}$,
    where $I = \lambda x. (x_1 \sharing{x_1}{x})$ and $\widetilde{x} = x_1 , x_2 ,x_3$.
    First, $M_0$ evolves into an intermediate substitution~\eqref{eq:lin_cons_sub1}.
    The bag can provide for all shared variables, so it then evolves into an explicit substitution~\eqref{eq:lin_cons_sub2}:
    \begin{align}
        M_0
        &\redd  (x_1 \!\bag{\!x_2 \!\bag{\!x_3\ \oneb\!} \!}  \sharing{\widetilde{x}}{x}) \esubst{ \bag{\fail^{\emptyset} , y , I}   }{x}
        \label{eq:lin_cons_sub1}
        \\
        &\red
        (x_1 \!\bag{\!x_2 \!\bag{\!x_3\ \oneb\!}  \!}) \linexsub{\! \bag{\!\fail^{\emptyset} , y , I\!} \!/ \widetilde{x} } = M
        \label{eq:lin_cons_sub2}
    \end{align}
    Since $ \headf{M} = x_1$, one of the three elements of the bag will be substituted.
    $M$ represents a non-deterministic choice between the following three reductions:
    \begin{align*}
        \mathbin{\rotatebox[origin=r]{30}{$\red$}} &~ (\fail^{\emptyset} \bag{x_2 \bag{x_3\ \oneb}  }) \linexsub{  \bag{y, I} /  x_2,x_3  }
        = N_1
        \\[-5pt]
        M~
        \redd &~ (y \bag{x_2 \bag{x_3\ \oneb}  }) \linexsub{ \bag{\fail^{\emptyset} , I} /  x_2,x_3  }
        = N_2
        \\[-3pt]
        \mathbin{\rotatebox[origin=r]{-30}{$\red$}} &~ (I \bag{x_2 \bag{x_3\ \oneb}}) \linexsub{ \bag{\fail^{\emptyset} , y}  /  x_2,x_3  }
        = N_3
        \tag*{\qedhere}
    \end{align*}
\end{example}

\subsection{Resource Control for \texorpdfstring{\lamcoldetshlin}{Lambda} via Intersection Types}\label{sec:lamTypes}

Our type system for   $\lamcoldetshlin$ is based on non-idempotent intersection types.
As in prior works~\cite{conf/fossacs/PaganiR10,DBLP:conf/birthday/BoudolL00}, intersection types account for available resources in bags, which are unordered and have all the same type.
Because we admit the term $\fail^{\tilde x}$ as typable, we say that our system enforces \emph{well-formedness} rather than \emph{well-typedness}.
As we will see, well-typed terms form the sub-class of well-formed terms that does not include $\fail^{\tilde x}$ (see the text after \Cref{t:lamSRShort}).

Strict types ($\sigma, \tau, \delta$) and multiset types ($\pi, \zeta$) are defined as follows:
\begin{align*}
    \sigma, \tau, \delta ::=~ &
    \unit \sepr \arrt{ \pi }{\sigma}
   \qquad  \pi, \zeta ::=
    \bigwedge_{i \in I} \sigma_i \sepr \omega
\end{align*}
\noindent
Given a non-empty $I$, multiset types $\bigwedge_{i\in I}\sigma_i$ are given to bags of size $|I|$.
This operator is associative, commutative,  and non-idempotent (i.e., $\sigma\wedge\sigma\neq \sigma$), with identity $\omega$.
Notation $\sigma^k$ stands for $\sigma \wedge \cdots \wedge \sigma$ ($k$ times, if $k>0$) or $\omega$ (if $k=0$).

Judgments have the form $\Gamma\wfdash M:\tau$, with contexts defined as follows:
\begin{align*}
    \Gamma,\Delta &::= \dash \sepr \Gamma , x:\pi  \sepr \Gamma, x:\sigma
\end{align*}
where $\dash$ denotes the empty context.
We write $\dom(\Gamma)$ for the set of variables in~$\Gamma$.
For $\Gamma, x:\pi$, we assume $x \not \in \dom(\Gamma)$.
To avoid ambiguities, we write $x:\sigma^1$  to denote that the assignment involves a multiset type, rather than a strict~type.
Given $\Gamma$, its \emph{core context} $\core{\Gamma}$ concerns variables with types different from $\omega$; it is defined as
$\core{\Gamma} = \{ x:\pi \in \Gamma \,|\, \pi \not = \omega\}$.

\begin{definition}[Well-formedness in $\lamcoldetshlin$]
    A  term $M$ is \emph{well-formed} if there exists a context $\Gamma$ and a  type  $\tau$ such that the rules in \Cref{fig:wfsh_ruleslin} entail $\Gamma \wfdash  M : \tau $.
\end{definition}

\begin{figure}[t]
    \begin{mdframed} \mysmall
        \begin{mathpar}
            \mprset{sep=1.8em}
            \inferrule[$\redlab{FS{:}var^{\ell}}$]{ }{
                 {x}: \sigma \wfdash  {x} : \sigma
            }
            \and
            \inferrule[$\redlab{FS{:}\oneb^{\ell}}$]{ }{
                \dash \wfdash \oneb : \omega
            }
            \and
            \inferrule[$\redlab{FS{:}bag^{\ell}}$]{
                 \Gamma \wfdash N : \sigma
                \\
                 \Delta \wfdash C : \sigma^k
            }{
                 \Gamma , \Delta \wfdash \bag{N}\cdot C:\sigma^{k+1}
            }
            \and
            \inferrule[$\redlab{FS{:}fail}$]{
                \dom(\core{\Gamma)} = \widetilde{x}
            }{
                \core{\Gamma} \wfdash  \fail^{\widetilde{x}} : \tau
            }
            \and
            \inferrule[$\redlab{FS{:}weak}$]{
                 \Gamma  \wfdash M : \tau
            }{
                \Gamma ,  {x}: \omega \wfdash M\sharing{}{x}: \tau
            }
            \and
            \inferrule[$\redlab{FS{:}shar}$]{
                \Gamma ,  {x}_1: \sigma, \ldots,  {x}_k: \sigma \wfdash M : \tau
                \\
                k \not = 0
            }{
                \Gamma ,  {x}: \sigma^{k} \wfdash M \sharing{{x}_1 , \ldots ,  {x}_k}{x}  : \tau
            }
            \and
            \inferrule[$\redlab{FS{:}abs\dash sh}$]{
                \Gamma ,  {x}: \sigma^k \wfdash M\sharing{\widetilde{x}}{x} : \tau
            }{
                 \Gamma \wfdash \lambda x . (M\sharing{\widetilde{x}}{x})  : \sigma^k \rightarrow \tau
            }
            \and
            \inferrule[$\redlab{FS{:}app}$]{
                \Gamma \wfdash M : \sigma^{j} \rightarrow \tau
                \\
                \Delta \wfdash C : \sigma^{k}
                \\
            }{
                 \Gamma , \Delta \wfdash M\ C : \tau
            }
            \and
            \inferrule[$\redlab{FS{:}Esub}$]{
                \Gamma ,  {x}: \sigma^{j} \wfdash M\sharing{\widetilde{x}}{x} : \tau
                \\\\
                \Delta \wfdash C : \sigma^{k}
            }{
                \Gamma , \Delta \wfdash (M\sharing{\widetilde{x}}{x})\esubst{ C }{ x }  : \tau
            }
            \and
            \inferrule[$\redlab{FS{:}Esub^{\ell}}$]{
                \Gamma  ,  x_1:\sigma, \cdots , x_k:\sigma \wfdash M : \tau
                \\\\
                \Delta \wfdash C : \sigma^k
            }{
                \Gamma , \Delta \wfdash M \linexsub{C /  x_1, \cdots , x_k} : \tau
            }
        \end{mathpar}
    \end{mdframed}
    \caption{Well-Formedness Rules for $\lamcoldetshlin$.}
    \label{fig:wfsh_ruleslin}
\end{figure}

\noindent
In \Cref{fig:wfsh_ruleslin},
Rule~$\redlab{FS:var^{\ell}}$ types variables.
Rule~$\redlab{FS:\oneb^{\ell}}$ types the empty bag with $\omega$.
Rule~$\redlab{FS:bag^{\ell}}$ types the concatenation of bags.
Rule~$\redlab{FS:\fail}$ types the term $\fail^{\widetilde{x}}$ with a strict type $\tau$, provided that the domain of the core context coincides with $\widetilde{x}$ (i.e., no  variable in $\widetilde{x}$ is typed with $\omega$).
Rule~$\redlab{FS:weak}$ types $M\sharing{}{x}$ by weakening the context with $x:\omega$.
Rule~$\redlab{FS:shar}$ types $M\sharing{\widetilde{x}}{x}$ with $\tau$, provided that there are assignments to the shared variables in~$\widetilde{x}$.

Rule~$\redlab{FS:abs{\dash}sh}$ types an abstraction $\lambda x.( M\sharing{\widetilde{x}}{x})$ with  $\sigma^k\to \tau$, provided that   $M\sharing{\widetilde{x}}{x}:\tau$ can be entailed from an assignment  $x:\sigma^k$.
Rule~\redlab{FS:app}  types $(M\ C)$, provided that $M$ has type $\sigma^j \to \tau$ and  $C$ has type $\sigma^k$. Note that, unlike usual intersection type systems, $j$ and $k$ may differ.
Rule~$\redlab{FS:Esub}$ types the intermediate substitution of a bag $C$ of type $\sigma^k$, provided that $x$ has type $\sigma^j$; again, $j$ and $k$ may differ.
Rule~$\redlab{FS:Esub^{\ell}}$  types  $M\linexsub{C/\widetilde{x}}$ as long as $C$ has type $\sigma^{|\widetilde{x}|}$,  and  each $x_i \in \widetilde{x}$ is of type $\sigma$.

Well-formed terms satisfy subject reduction (SR), whereas \emph{well-typed} terms, defined below, satisfy also subject expansion (SE).
\ifaplas
See~\cite{DBLP:journals/corr/abs-2205-00680} for details.
\else
See \secref{a:lamTypes} and \secref{a:subexpan} for details.
\fi

\begin{theorem}[SR in \lamcoldetshlin]\label{t:lamSRShort}
    If $ \Gamma \wfdash M:\tau$ and $M \redd M'$, then $ \Gamma \wfdash M' :\tau$.
\end{theorem}


From our system for well-formedness we can extract a system for \emph{well-typed} terms, which do not include $\fail^{\tilde x}$.
Judgments for well-typedness are denoted $ \Gamma \wtdash M:\tau$, with rules copied from \Cref{fig:wfsh_ruleslin} (the rule name prefix \texttt{FS} is replaced with \texttt{TS}), with the following modifications:
(i)~Rule~$\redlab{TS{:}fail}$ is removed; (ii)~Rules~$\redlab{TS{:}app}$ and $\redlab{TS{:}Esub}$ are modified to disallow a mismatch between variables and resources, i.e., multiset types should match in size.
Well-typed terms are also well-formed, and thus satisfy SR.
Moreover, \myrev{as a consequence of adopting (non-idempotent) intersection types,} they also satisfy~SE:

\begin{theorem}[SE in $\lamcoldetshlin$]\label{t:lamSEShort}
    If $ \Gamma \wtdash M':\tau$ and $M \redd M'$, then $ \Gamma \wtdash M :\tau$.
\end{theorem}


\section{A Typed Translation of \texorpdfstring{\lamcoldetshlin}{Lambda} into \texorpdfstring{\clpi}{Pi}
}
\label{s:trans}

Up to here, we have presented \clpi and \lamcoldetshlin, two calculi with  \nond.
While \clpi features \nondt choice,
\lamcoldetshlin
\myrev{is a prototypical programming language in which implicit non-determinism}
implements  fetching of resources.
They  control resources using different type systems  (session types in \clpi, intersection types  in \lamcoldetshlin).
To reconcile these differences and
\myrev{illustrate the potential of \clpi to precisely model non-determinism as found in realistic programs/protocols},
we give a translation of  \lamcoldetshlin into \clpi. 
This translation 
preserves types (\thmref{def:encod_judge}) and respects well-known  criteria for dynamic correctness~\cite{DBLP:journals/iandc/Gorla10,DBLP:phd/dnb/Peters12a,DBLP:journals/corr/abs-1908-08633} (\thmref{t:correncLazy}).

\paragraph*{The Translation.}
Given a \lamcoldetshlin-term \(M\), its translation into \clpi is denoted \(\piencod{M}_{u}\) and given in \Cref{fig:encodinglin}.
As usual, every variable \(x\) in $M$ becomes a name \(x\) in process \(\piencod{M}_{u}\), where
name \(u\) provides the behavior of \(M\).
A peculiarity is that, to handle failures in \lamcoldetshlin, \(u\) is a non-deterministically available session: the translated term can be available or not, as signaled by prefixes \(\pSome u\) and \(\pNone u\), respectively.
As a result, reductions from \(\piencod{M}_{u}\) include synchronizations that codify $M$'s behavior but also   synchronizations that confirm a session's availability.

\begin{figure*}[t]
    \begin{mdframed} \mysmall
        \begin{align*}
            \piencodf{ {x}}_u
            &= \pSome x; \pFwd [x<>u]
            \\
            \piencodf{\lambda x.M}_u
            &= \pSome u;\pIn u(x); \piencodf{M}_u
            \\
            \piencodf{(M\,C)}_u
            &=
            \pRes{v} (\piencodf{M}_v \| \pSomeIn v{u , \llfv{C}};\pOut* v[x]; ( \piencodf{C}_x   \| \pFwd [v<>u]  ) )
            \\
            \piencodf{ M \esubst{ C }{ x} }_u
            &=  \pRes{x}( \piencodf{ M}_u \| \piencodf{ C }_x )
            \\
            \piencodf{\bag{N} \cdot~ C}_{{x}}
            &=
            \pSomeIn {x}{\lfv{C}, \lfv{N} }; \pIn x(y_i); \pSomeIn {x}{y_i, \lfv{C}, \lfv{N}}; \pSome {x}; \pOut* {x}[z_i];
            \\ &\qquad
            ( \pSomeIn z_i{\lfv{N}}; \piencodf{N}_{z_i}
              \| \piencodf{C}_{{x}} \| \pNone y_i )
            \\
            \piencodf{{\oneb}}_{{x}}
            &=
            \pSomeIn {x}{\emptyset};\pIn x(y_n);  ( \pSome  y_n; \pClose y_n[]  \| \pSomeIn {x}{\emptyset}; \pNone {x} )
            \\
            \piencodf{M \!\linexsub{ \!\bag{ \!N_1 , N_2\! } \!/ x_1,x_2 }}_u
            &=
            \pRes{z_1}( \pSomeIn z_1{\lfv{N_{1}}};\piencodf{ N_{1} }_{ {z_1}} \| \pRes{z_2} (
                \pSomeIn z_2{\lfv{N_{2}}};\piencodf{ N_{2} }_{ {z_2}}
            \\ & \qquad
                \| \bignd_{x_{i} \in \{ x_1 , x_2  \}} \bignd_{x_{j} \in \{ x_1, x_2 \setminus x_{i}  \}} \piencodf{ M }_u \{ z_1 / x_{i} \} \{ z_2 / x_{j} \} ) )
            \\
            \piencodf{M\sharing{}{x}}_u
            &=
            \pSome {x}; \pOut* {x}[y_i]; ( \pSomeIn y_i{ u , \lfv{M} }; \pWait  y_{i() } ;\piencodf{M}_u \| \pNone  {x } )
            \\
            \piencodf{M\sharing{\widetilde{x}}{x}}_u
            &=
                \pSome {x}; \pOut* {x}[y_i]; \big( \pSomeIn y_i{ \emptyset }; \pWait  y_{i() } ; \0
                \| \pSome {x}; \pSomeIn {x}{u, \lfv{M} \setminus  \widetilde{x} };
                \\ & \qquad
            \bignd_{x_i \in \widetilde{x}} \,\pIn {x}({x)_i};\piencodf{M\sharing{(\widetilde{x} {\setminus} x_i)}{x}}_u \big)
            \\
            \piencodf{\fail^{x_1, \ldots, x_k}}_u
            &= \pNone  u  \| \pNone  x_1 \| \ldots \| \pNone  x_k
        \end{align*}
    \end{mdframed}
    \caption{Translation of $\lamcoldetshlin$ into \clpi.}\label{fig:encodinglin}
\end{figure*}
At its core, our translation  follows Milner's.
This way, e.g., the process \(\piencodf{(\lambda x.M)\ C}_u\)  enables synchronizations between \(\piencodf{\lambda x.M}_v\) and \(\piencodf{C}_{x}\) along
  name~$v$, resulting in the translation of an intermediate substitution.
The \emph{key novelty} is the role and treatment of   non-determinism. Accommodating \col \nond is non-trivial, as it entails translating explicit substitutions and sharing in $\lamcoldetshlin$ using the \nondt choice operator $\nd$ in \clpi.
Next we discuss these novel aspects, while highlighting differences with respect to a translation by Paulus et al.~\cite{conf/fscd/PaulusNP21}, which is given in the confluent setting (see \secref{s:disc}).

   In \Cref{fig:encodinglin}, non-deterministic choices occur only in two places: the translations of   $M\linexsub{C/\widetilde{x}}$ (explicit substitutions) and  $M\sharing{\widetilde{x}}{x}$ (non-empty sharing).
    Roughly speaking, the position  of $\nd$  in the translation of $M\linexsub{C/\widetilde{x}}$  represents the most desirable way of mimicking the fetching of terms from a bag.
This use of $\nd$ is a central idea in our translation: as we explain below,  it allows for appropriate commitment in  \nondt choices, but also for \emph{delayed} commitment when necessary.

For simplicity, we consider explicit substitutions \(M\mkern-3mu\linexsub{C/\widetilde{x}}\) where $C=\bag{\mkern-2mu N_1{,}N_2}$ and \(\widetilde{x}=x_1,x_2\).
The translation \(\piencodf{M \linexsub{C/\widetilde{x}}}_{u}\) uses the processes \(\piencodf{N_i}_{z_i}\), where $z_1, z_2$ are fresh.
First, each bag item confirms its behavior.
Then, a variable~\(x_{i} \in \widetilde{x}\) is chosen non-deterministically;
we ensure that these choices consider all variables.
Note that writing \(\bignd_{x_{i} \in \{x_1,x_2 \}}\bignd_{x_{j}\in \{x_1,x_2\}\setminus x_{i} }\) is equivalent to non-de\-ter\-mi\-nis\-tic\-ally assigning \(x_{i},x_{j} \) to each permutation of \(x_1,x_2\).
The resulting choice involves \(\piencodf{M}_{u}\) with $x_{i}, x_j$ substituted by $z_1, z_2$.
Commitment here is triggered only via synchronizations along $z_1$ or $z_2$; synchronizing with
$\pSomeIn z_i{\lfv{N_{i}}};\piencodf{N_{i} }_{ {z_i}}$ then represents fetching   $N_{i}$ from the bag.

The process
\(\piencodf{M\sharing{\widetilde{x}}{x}}_{u}\)
proceeds as follows.
First, it confirms its behavior along \(x\).
Then it sends a name \(y_i\) on \(x\), on which a failed reduction may be handled.
Next,
the translation confirms again its behavior along \(x\) and non-deterministically receives a reference to an $x_i \in \widetilde{x}$.
Each branch consists of $\piencodf{M\sharing{(\widetilde{x} {\setminus} x_i)}{x}}_u$.
The possible choices are permuted, represented by~\(\bignd_{x_i\in \widetilde{x}}\).
{Notice that synchronizations with $\piencodf{M\sharing{(\widetilde{x} {\setminus} x_i)}{x}}_u$ and bags delay commitment in this choice (we return to this point below, when discussing correctness).}
The process \(\piencodf{M\sharing{}{x}}_{u}\) is similar but simpler: in this case the  name \(x\) fails, as it cannot take further elements to substitute.

In case of a failure (i.e., a mismatch between the size of the bag~\(C\) and the number of variables in $M$), our translation ensures that the confirmations of \(C\) will not succeed.
This is how failure in \lamcoldetshlin is correctly translated to failure in~\clpi.

\paragraph*{Translation Correctness.}
The translation is typed: intersection types in \lamcoldetshlin are translated into session types in \clpi (\Cref{fig:enc_typeslin}).
This translation of types abstractly describes how non-deterministic fetches are codified as non-deterministic sessions.
{
    It is worth noting that this translation of types is the same as in~\cite{conf/fscd/PaulusNP21}.
   This is not surprising: as we have seen, session types effectively abstract away from the behavior of processes, as all branches of a non-deterministic choice use the same typing context.
    Still, it is pleasant that the translation of types remains unchanged across different translations with
        our (non-confluent) non-determinism (in \Cref{fig:encodinglin})
        and with
    confluent non-determinism (in~\cite{conf/fscd/PaulusNP21}).
}

To state \emph{static} correctness,
we require the following definition:

\begin{definition}\label{def:enc_sestypfaillin}
    Let
    $
        \Gamma =  {{x}_1: \sigma_1}, \ndots,  {{x}_m : \sigma_m},  {{v}_1: \pi_1} , \ndots ,  {v}_n: \pi_n
    $
    be a context.
    The translation  $\piencodf{\cdot}_{\_}$  in~\Cref{fig:enc_typeslin} extends to contexts
    as follows:
    \begin{align*}
        \piencodf{\Gamma} = {} & {x}_1 : \& \overline{\piencodf{\sigma_1}} , \cdots ,   {x}_m : \& \overline{\piencodf{\sigma_m}} ,
         {v}_1:  \overline{\piencodf{\pi_1}_{(\sigma, i_1)}}, \cdots ,  {v}_n: \overline{\piencodf{\pi_n}_{(\sigma, i_n)}}
    \end{align*}
\end{definition}

\begin{figure}[t]
    \begin{mdframed} \mysmall
        \begin{align*}
            \piencodf{\unit} &= \& \onef
            &
            \qquad \qquad
            \piencodf{\sigma^{k}   \rightarrow \tau} &= \&( \dual{\piencodf{ \sigma^{k}  }_{(\sigma, i)}} \ampy \piencodf{\tau})
            \\
            \piencodf{ \sigma \wedge \pi }_{(\tau, i)} = \oplus(( \& \onef) \ampy ( \oplus  \& (( \oplus \piencodf{\sigma} ) \otimes (\piencodf{\pi}_{(\tau, i)}))))
            \span\span\span
            \\
            \piencodf{\omega}_{(\sigma, i)}
            &= \begin{cases}
                \oplus ((\& \1) \parr (\oplus \& \1))
                & \text{if $i = 0$}
                \\
                \oplus ((\& \1) \parr (\oplus\, \& ((\oplus \piencodf{\sigma}) \tensor (\piencodf{\omega}_{(\sigma, i-1)}))))
                & \text{if $i > 0$}
            \end{cases}
            \span\span
        \end{align*}
    \end{mdframed}
    \caption{Translation of intersection types into session types  (cf.\ \defref{def:enc_sestypfaillin}).}
    \label{fig:enc_typeslin}
\end{figure}

\noindent
We then have
our type preservation property: well-formed terms translate into well-typed processes.

\begin{theorem}
    \label{def:encod_judge}
    If
    $ \Gamma \wfdash {M} : \tau$,
    then
    $
        \piencodf{{M}}_u \vdash
        \piencodf{\Gamma},
        u : \piencodf{\tau}
    $.
\end{theorem}


To state \emph{dynamic} correctness, we rely on established (abstract) notions to characterize \emph{correct translations}.
A language \({\cal L}=(L,\to)\)  consists of a set of terms $L$ and a reduction relation $\to$ on $L$.
Each language ${\cal L}$ is assumed to contain a success constructor \(\sucs{}\).
A term \(T \in  L\) has {\em success}, denoted \(\succp{T}{\sucs{}}\), when there is a sequence of reductions (using \(\to\)) from \(T\) to a term satisfying success criteria.

Given ${\cal L}_1=(L_1,\to_1)$ and ${\cal L}_2=(L_2,\to_2)$, we seek translations \(\encod{\cdot }{}: {L}_1 \to {L}_2\) that are correct: they satisfy well-known correctness criteria~\cite{DBLP:journals/iandc/Gorla10,DBLP:phd/dnb/Peters12a,DBLP:journals/corr/abs-1908-08633}.
The following describes a set of correctness criteria that together determine the correctness of a translation.

\begin{definition}[Correct Translations]\label{d:encCriteria}
    Let $\mathcal{L}_1= (\mathcal{M}, \shred_1)$ and  $\mathcal{L}_2=(\mathcal{P}, \shred_2)$ be two languages.
    Let $ \asymp_2 $ be an equivalence 
    defined over $\mcl{L}_2$.
    We use $M,M',\ldots$ (resp.\ $P,P',\ldots$) to range over elements in $\mcl{M}$  (resp.\ $\mcl{P}$).
    The following are  criteria for a translation $\encod{\cdot }{}: {\cal M}\to {\cal P}$:
    \begin{description}
        \item \textbf{Completeness:}
            For every ${M}, {M}' $ such that ${M} \shred_1^\ast {M}'$, there exists $ P $ such that $ \encod{{M}}{} \shred_2^\ast P \asymp_2 \encod{{M}'}{}$.

        %
        \item \textbf{Weak Soundness:}
            For every $M$ and $P$ such that $\encod{M}{} \shred_2^\ast P$, there exist  $M'$, $P' $ such that $M \shred_1^\ast M'$ and $P \shred_2^\ast P' \asymp_2 \encod{M'}{}  $.

        %
        \item \textbf{Success Sensitivity:}
            For every ${M}$, we have $\succp{M}{\sucs{}}$ if and only if $\succp{\encod{M}{}}{\sucs{}}$.
    \end{description}
    %
\end{definition}

Let us write $\Lambda$ to denote the set of all well-formed \lamcoldetshlin terms, and $\Pi$ for the set of all well-typed \clpi processes, both including $\sucs{}$.
We have the following:

\begin{theorem}[Translation correctness under $\reddLazy$]
    \label{t:correncLazy}
    The translation $\piencod{\cdot}_{\_}: (\Lambda, \red) \to (\Pi, \reddLazy)$ is 
    correct  (cf.\ \Cref{d:encCriteria}) using equivalence $\equiv$ (\Cref{f:pilang}).
\end{theorem}

\noindent
The proof of \Cref{t:correncLazy} involves instantiating/proving each of the parts of \mbox{\defref{d:encCriteria}}.
%
%
%
%
%
%
%
%
%
%
Among these, \emph{weak soundness} is the most  challenging to prove.
Prior work on translations of typed $\lambda$ into $\pi$ with confluent non-determinism~\cite{conf/fscd/PaulusNP21} rely critically on
confluence to match a behavior in $\pi$ with a corresponding behavior in $\lambda$.
Because in our setting  confluence is lost,  we must resort to a different proof.

As already discussed, our translation makes the implicit non-determinism in a $\lamcoldetshlin$-term $M$ explicit by adding non-deterministic choices in key points of $\piencodf{M}_u$.
  Our reduction $\reddLazy$ preserves those branches that simultaneously have the same prefix available (up to $\prefRel$). In proving
  weak soundness,
  we exploit the fact that reduction entails delayed commitment.
    To see this, consider the following terms:
\begin{eqnarray}
    \pRes{x}(  (\alpha_1 ; P_1 \nd \alpha_2 ; P_2) \| Q)  \label{ex:sound1}
    \\
    \pRes{x}(  \alpha_1 ; P_1 \| Q) \nd \pRes{x}(\alpha_2 ; P_2 \| Q) \label{ex:sound2}
\end{eqnarray}
In \eqref{ex:sound1}, commitment to a choice
depends on whether $\alpha_1 \prefRel\alpha_2$ holds.
If $\alpha_1 \not \prefRel \alpha_2$, a choice is made; otherwise,   commitment is delayed, and depends on $P_1$ and $P_2$.
Hence, in \eqref{ex:sound1} the possibility of committing to either branch is kept open.
In contrast, in \eqref{ex:sound2} commitment to a choice is independent of   $\alpha_1 \prefRel\alpha_2$.

Our translation exploits the delayed commitment of non-determinism illustrated by \eqref{ex:sound1} to mimic commitment to non-deterministic choices in \lamcoldetshlin, which manifests in fetching resources from bags.
The fact that this delayed commitment preserves information about the different branches (e.g., $P_1$ and $P_2$ in \eqref{ex:sound1}) is essential in characterizing and establishing
weak soundness,
i.e., in matching a target behavior in \clpi with a corresponding source behavior in \lamcoldetshlin.
In contrast, forms of non-determinism in $\piencodf{N}_u$ that resemble \eqref{ex:sound2} are very useful to characterize behaviors in \lamcoldetshlin different from fetching.

\section{Summary and Related Work}
\label{s:disc}

We studied the interplay between resource control and non-determinism in typed \ifthesis\else functional and \fi concurrent calculi.
The interplay is subtle and non-trivial: committing non-determinism is at odds with linearity-based resource control.

We introduced \clpi\ifthesis, a calculus \else and \lamcoldetshlin, two calculi \fi with \nond, \ifthesis with a type system \else both with type systems \fi for resource control.
Inspired by the untyped \picalc, \nond in \clpi is lazy and explicit, with  session types defined following `propositions-as-sessions'~\cite{conf/esop/CairesP17}.
\ifthesis
We defined an alternative eager semantics for \clpi, and compared it with the lazy semantics.
\else
In \lamcoldetshlin, \col \nond  arises in the fetching of resources, and is regulated by intersection types.
A correct translation of \lamcoldetshlin into \clpi precisely connects their different forms of  \nond, preserving types; it contributes to works on relative expressiveness in concurrency in general, and to `functions-as-processes' in particular.
\fi

\paragraph*{Related Work.}

\ifthesis
A central theme in our chapter is the expressivity of non-deterministic choice   in the \picalc, a relevant line of work pursued by Palamidessi~\cite{journal/mscs/Palamidessi03}, Nestmann~\cite{journal/ic/Nestmann00}, Peters~\cite{thesis/Peters12}, among others.
These works study the expressivity of  forms of choice (e.g., mixed, separate) and their connection with (a)synchronous communication, sticking to the usual semantics for non-determinism (cf.\ \eqref{eq:clpi:usualpi}).
In contrast, we study how to express commitment within a typed setting.
\fi

Integrating \nond within session types is non-trivial, as carelessly discarding branches would break linearity (and typability).
The work by Caires and Pérez~\cite{conf/esop/CairesP17}, already mentioned, develops a confluent semantics by requiring that non-determinism is only used inside the monad $\& A$; our non-confluent semantics drops this requirement.
This allows us to consider non-deterministic choices not possible in~\cite{conf/esop/CairesP17}, such as, e.g., selections of different labels.
We stress that linearity is not jeopardized: the branches of `$\nd$' do not represent \emph{different sessions}, but \emph{different implementations} of the same sessions.

Atkey et al.~\cite{chapter/AtkeyLM16} and Kokke et al.~\cite{journal/lmcs/KokkeMW20} extend `pro\-po\-si\-tions-as-sessions' with \nond.
Their approaches are very different (conflation of the additives and bounded linear logic, respectively) and support non-determinism for \emph{unrestricted names only}.
This is a major difference:
\Cref{x:clpi:possible,x:clpi:precongLazy,x:clpi:safeOrRisk} can only be supported in \cite{chapter/AtkeyLM16,journal/lmcs/KokkeMW20} by dropping linearity, which is important there.
\ifthesis\else
Also, \cite{chapter/AtkeyLM16,journal/lmcs/KokkeMW20} do not connect with typed \lamcalci, as we do. 
\fi
Rocha and Caires also consider non-determinism, relying on confluence in~\cite{conf/icfp/RochaC21} and on unrestricted names in~\cite{conf/esop/RochaC23}.

Casal et al.~\cite{journal/tcs/CasalMV22,conf/esop/VasconcelosCAM20} develop a type system for \emph{mixed sessions} (sessions with mixed choices), which does not ensure deadlock-freedom.
Deadlock-freedom by typing is hard, especially when session names can be passed around (delegation), as in \clpi; an advantage of `propositions-as-sessions' is that deadlock-freedom follows directly from the link between cut-elimination and process communication.

Different non-deterministic \lamcalci exist; see, e.g.,~\cite{journal/tcs/EhrhardR03,conf/fossacs/PaganiR10,DBLP:journals/iandc/Boudol94}.
Our language \lamcoldetshlin is most related to resource calculi by Boudol~\cite{DBLP:conf/concur/Boudol93}, Boudol and Laneve~\cite{DBLP:conf/birthday/BoudolL00}, and by Pagani and {Ronchi Della Rocca}~\cite{conf/fossacs/PaganiR10}.
The calculi in~\cite{DBLP:conf/concur/Boudol93,DBLP:conf/birthday/BoudolL00} are related to \lamcoldetshlin as their \nond is committing and implicit. In these two works linear resources can be consumed \emph{at most} once, rather than \emph{exactly} once.
The work~\cite{conf/fossacs/PaganiR10} considers non-committing \nond that is both implicit (i.e., fetching as in \lamcoldetshlin) and explicit (via a sum operator on terms).
Both~\cite{DBLP:conf/concur/Boudol93} and~\cite{conf/fossacs/PaganiR10} develop (non-idempotent) intersection type systems to regulate resources.
Our type system enforces that all terms in a bag have the same type; the system in~\cite{conf/fossacs/PaganiR10} does not enforce this condition.
Unlike these type systems, our system for well-formedness can type $\lamcoldetshlin$-terms with a lack or an excess of resources.

 Boudol and Laneve~\cite{DBLP:conf/birthday/BoudolL00} and Paulus et al.~\cite{conf/fscd/PaulusNP21} give translations of resource \lamcalci into $\pi$.
The translation in~\cite{DBLP:conf/birthday/BoudolL00}
is used to study the semantics induced upon $\lambda$-terms by translation into \(\pi\); unlike ours, it does not consider types.
As already mentioned in~\secref{s:trans},  Paulus et al.~\cite{conf/fscd/PaulusNP21} relate calculi with \emph{confluent} \nond: a resource \lamcalc with sums on terms, and the session $\pi$-calculus from~\cite{conf/esop/CairesP17}.
Although we translate intersection types into session types exactly as they do,
the translation of terms is very different:
while here we use non-deterministic choice to mimic the sharing construct, the translation in~\cite{conf/fscd/PaulusNP21} uses it to translate bags.
Hence, our \Cref{t:correncLazy} cannot be derived from~\cite{conf/fscd/PaulusNP21}.

Milner's  seminal work~\cite{conf/icalp/Milner90} connects
untyped, deterministic $\lambda$ and $\pi$.
Sangiorgi studies the behavioral equivalences that Milner's translations induce on $\lambda$-terms~\cite{DBLP:phd/ethos/Sangiorgi93,DBLP:journals/iandc/Sangiorgi94,DBLP:journals/mscs/Sangiorgi99};   in particular, the work~\cite{DBLP:journals/iandc/Sangiorgi94} considers an untyped, non-de\-ter\-minis\-tic \lamcalc.
Sangiorgi and Walker~\cite{book/SangiorgiW03} offer a unified presentation of (typed) translations of $\lambda$ into $\pi$: 
they consider simply-typed \lamcalci and $\pi$-calculi with input-output types;  non-deterministic choices are not considered.

The last decade of work on `prop\-o\-sitions-as-sessions' has delivered   insightful connections with typed \lamcalci---see, e.g.,~\cite{conf/icfp/Wadler12,DBLP:conf/fossacs/ToninhoCP12,DBLP:conf/esop/ToninhoY18}.
Excepting~\cite{conf/fscd/PaulusNP21}, already discussed,
none of these works consider non-deterministic \lamcalci.

\paragraph{Acknowledgments}
We are grateful to the anonymous reviewers for useful remarks.
This research has been supported by the Dutch Research Council (NWO) under project No.\ 016.Vidi.189.046 (Unifying Correctness for Communicating Software) 
and by the EPSRC Fellowship `VeTSpec: Verified Trustworthy Software Specification' (EP/R034567/1).

\addcontentsline{toc}{section}{References}
\bibliography{references}

\ifaplas


\title{Monitoring Blackbox Implementations \texorpdfstring{\\}{} of  Multiparty Session Protocols\thanks{%
This research has been supported by the Dutch Research Council (NWO) under project No.\ 016.Vidi.189.046 (Unifying
Correctness for Communicating Software).}}
\titlerunning{Monitoring Blackbox Implementations of Multiparty Session Types}
\author{Bas van den Heuvel\orcidID{0000-0002-8264-7371} \and Jorge A.\ Pérez\orcidID{0000-0002-1452-6180} \and Rares A.\ Dobre}
\institute{University of Groningen, Groningen, The Netherlands}

\maketitle

\begin{abstract}
    We present a framework for the distributed monitoring of networks of components that coordinate by  message-passing, following multiparty session protocols specified as global types.
   We improve over prior works by (i)~supporting components whose exact specification is unknown (``blackboxes'') and (ii)~covering protocols that cannot be analyzed by existing techniques.
    We first give a procedure for synthesizing monitors for blackboxes from global types, and precisely define when a  blackbox correctly \emph{satisfies} its global type.
    Then, we prove that  monitored blackboxes are  \emph{sound} (they correctly follow the protocol) and \emph{transparent} (blackboxes with and without monitors are behaviorally equivalent).
    \keywords{distributed monitoring \and message-passing \and concurrency \and  multiparty session types }
\end{abstract}

\section{Introduction}
\label{s:introduction}

Runtime verification excels at analyzing systems with components that cannot be (statically) checked, such as closed-source and third-party components with unknown/partial specifications~\cite{book/BartocciFFR18,book/FrancalanzaPS18}.
In this spirit, we present a monitoring framework for networks of  communicating components.
We adopt \emph{global types} from \emph{multiparty session types}~\cite{conf/popl/HondaYC08,journal/acm/HondaYC16} both to specify protocols and to synthesize monitors.
As we explain next, rather than process implementations, we consider \emph{``blackboxes''}---components whose exact structure is unknown.
Also, aiming at wide applicability, we cover networks of monitored components that implement global types that go beyond the scope of existing techniques.

Session types provide precise specifications of the protocols that components should respect.
It is then natural to use session types as references for distributed monitoring~\cite{journal/tcs/BocchiCDHY17,conf/ecoop/BartoloBurloFS21,conf/popl/JiaGP16,journal/jlamp/GommerstadtJP22,conf/tgc/Thiemann14,conf/icfp/IgarashiTVW17}.
In particular, Bocchi \etal~\cite{journal/tcs/BocchiCDHY17,conf/forte/BocchiCDHY13,conf/tgc/ChenBDHY11} use multiparty session types to monitor networks of \picalc processes.
Leveraging notions originally conceived for static verification (such as \emph{global types} and their \emph{projection} onto \emph{local types}), their framework guarantees the correctness of monitored networks with statically and dynamically checked components.

However, existing monitoring techniques based on multiparty session types have two limitations.
One concerns the class of protocols they support; the other is their reliance on fully specified components, usually given as (typed) processes.
That is, definitions of networks assume that a component can be inspected---an overly strong assumption in many cases.
There is then a tension between (i)~the assumptions on component structure and (ii)~the strength of formal guarantees: the less we know about components, the harder it is to establish such guarantees.

\paragraph{Our approach.}

We introduce a new approach to monitoring based on multiparty session types that relies on minimal assumptions on a component's structure.
As key novelty, we consider \emph{blackboxes}---components with unknown structure but observable behavior---and \emph{networks} of monitored blackboxes that use asynchronous message-passing to implement multiparty protocols specified as global types.

As running example, let us consider the global type $\gtc{G_{\sff{a}}}$ (inspired by an example by Scalas and Yoshida~\cite{conf/popl/ScalasY19}), which expresses an authorization protocol between three participants: server ($s$), client ($c$), and authorization service ($a$).
\begin{equation}
    \gtc{G_{\sff{a}}} \deq \mu X \gtc. s {!} c \gtBraces{ \msg \sff{login}<> \gtc. c {!} a \gtBraces{ \msg \sff{pwd}<\sff{str}> \gtc. a {!} s \gtBraces{ \msg \sff{succ}<\sff{bool}> \gtc. X } } \gtc, \msg \sff{quit}<> \gtc. \tEnd }
    \label{eq:mpstMon:auth}
\end{equation}
This recursive global type ($\mu X$) specifies that $s$ sends to $c$ ($s {!} c$) a choice between labels $\sff{login}$ and $\sff{quit}$.
In the $\sff{login}$-branch, $c$ sends to $a$ a password of type $\<\sff{str}\>$ and $a$ notifies $s$ whether it was correct, after which the protocol repeats ($X$).
In the $\sff{quit}$-branch, the protocol ends~($\tEnd$).
As explained in~\cite{conf/popl/ScalasY19}, $\gtc{G_{\sff{a}}}$ is not supported by most theories of multiparty sessions, including those in~\cite{journal/tcs/BocchiCDHY17,conf/forte/BocchiCDHY13,conf/tgc/ChenBDHY11}.

\begin{figure}[t]
    \begin{center}
        \begin{tikzpicture}[every loop/.style={min distance=5mm, looseness=10}]
            \begin{scope}[local bounding box=LTSc, node distance=4mm, every node/.style={inner sep=.5mm}]
                \node (n0) {$P_c$};
                \node [below=of n0] (n1) {};
                \node [right=of n1] (n2) {};
                \node [below=of n2] (n3) {};
                \draw [-to] (n0) to (n1);
                \draw [-to] (n1) to (n3);
                \draw [-to] (n0) to (n2);
                \draw [-to] (n2) to (n3);
                \draw [-to] (n2) to [loop right] ();
            \end{scope}
            \node [draw,fit=(LTSc), inner sep=.5mm] (LTScBB) {};
            \node [draw,right=-.4pt of LTScBB.north east, anchor=north west] (Mc) {$\mc{M_c}$};

            \node [draw,right=2cm of Mc.east] (Ms) {$\mc{M_s}$};
            \begin{scope}[local bounding box=LTSs, shift={($(Ms.north east)+(2.5mm+.5pt,-2.5mm+.5pt)$)}, node distance=4mm, every node/.style={inner sep=.5mm}]
                \node (n0) {$P_s$};
                \node [below=of n0] (n1) {};
                \node [right=of n1] (n2) {};
                \node [below=of n1] (n3) {};
                \draw [-to] (n0) to (n1);
                \draw [-to] (n1) to (n2);
                \draw [-to] (n0) to (n2);
                \draw [-to] (n1) to (n3);
                \draw [-to] (n3) to [loop right] ();
            \end{scope}
            \node [draw,fit=(LTSs),inner sep=.5mm] (LTSsBB) {};

            \node [right=1cm of Mc.east, anchor=center, yshift=-1.0cm] (G) {$\gtc{G_{\sff{a}}}$};

            \node [draw, below=6.5mm of G.south, anchor=north] (Ma) {$\mc{M_a}$};
            \begin{scope}[local bounding box=LTSa, shift={($(Ma.south west)+(2.7mm+.5pt,-2.4mm+.5pt)$)}, node distance=4mm, every node/.style={inner sep=.5mm}]
                \node (n0) {$P_a$};
                \node [right=of n0] (n2) {};
                \node [below=of n2] (n1) {};
                \node [right=of n2] (n4) {};
                \node [below=of n4] (n3) {};
                \draw [-to] (n0) to (n2);
                \draw [-to] (n2) to (n4);
                \draw [-to] (n4) to (n3);
                \draw [-to] (n3) to (n1);
                \draw [-to] (n1) to (n2);
                \node [shift={($(n3.south east)+(1pt,-1pt)$)}] {};
            \end{scope}
            \node [draw,fit=(LTSa), inner sep=.5mm] (LTSaBB) {};

            \node [left=8mm of G, yshift=-2mm, anchor=center, inner sep=0mm] (Rca) {$\rtc{R_{c,a}}$};
            \node [right=8mm of G, yshift=-2mm, anchor=center, inner sep=0mm] (Rsa) {$\rtc{R_{s,a}}$};
            \node [above=8.0mm of G, anchor=center] (Rcs) {$\rtc{R_{c,s}}$};

            \draw [gray,-to,line width=1pt] (G) to (Rca);
            \draw [gray,-to,line width=1pt] (G) to (Rsa);
            \draw [gray,-to,line width=1pt] (G) to (Rcs);

            \draw [-to,line width=2pt] (G) to (Mc);
            \draw [-to,line width=2pt] (G) to (Ms);
            \draw [-to,line width=2pt] (G) to (Ma);

            \draw [-to,dotted] ($(Rca.south)+(2mm,-.5mm)$) to (Ma.north west);
            \draw [-to,dotted] ($(Rsa.south)+(-2mm,-.5mm)$) to (Ma.north east);
            \draw [-to,dotted] ($(Rca.north)+(0mm,1mm)$) to (Mc);
            \draw [-to,dotted] (Rcs) to (Mc);
            \draw [-to,dotted] ($(Rsa.north)+(0mm,1mm)$) to (Ms);
            \draw [-to,dotted] (Rcs) to (Ms);
        \end{tikzpicture}%
    \end{center}

    \caption{
        Monitoring setup based on the global type (multiparty protocol) $\gtc{G_{\sff{a}}}$\ifmpstMon~\eqref{eq:mpstMon:auth}\fi.
        Each protocol participant has a blackbox (an \LTS), attached to a monitor (e.g. $P_c$ and~$\mc{M_c}$).
        The monitors are synthesized from $\gtc{G_{\sff{a}}}$ (thick arrows).
        Relative types (e.g. $\rtc{R_{c,s}}$) obtained by projection from $\gtc{G_{\sff{a}}}$ (thin gray arrows) are used in this synthesis (dotted arrows).
    }
    \ifmpstMon \label{f:mpstMon:setup}
    \else \label{f:intro:mpstMon:diag} \fi
\end{figure}

\Cref{f:mpstMon:setup} illustrates our approach to monitoring global types such as $\gtc{G_{\sff{a}}}$.
There is a blackbox per participant, denoted $P_s$, $P_c$, and $P_a$, whose behavior is given by a labeled transition system (\LTS).
Each blackbox implements a participant as dictated by $\gtc{G_{\sff{a}}}$
while coupled with a monitor ($\mc{M_s}$, $\mc{M_c}$, and $\mc{M_a}$ in \Cref{f:mpstMon:setup}).
Monitors are \emph{synthesized} from $\gtc{G_{\sff{a}}}$ by relying on \emph{relative types}~\cite{journal/scico/vdHeuvelP22}\ifmpstMon\ (discussed in \Cref{c:relProj})\fi, which provide local views of the global type: they
specify protocols between \emph{pairs} of participants; hence, in the case of $\gtc{G_{\sff{a}}}$, we have three relative types:  $\rtc{R_{c,s}}$, $\rtc{R_{c,a}}$, and $\rtc{R_{s,a}}$.

Introduced in~\cite{journal/scico/vdHeuvelP22} for type-checking communicating components, relative types are instrumental to our approach.
They give a fine-grained view of protocols that is convenient for monitor synthesis.
Relative types explicitly specify \emph{dependencies} between participants, e.g., when the behavior of a participant $p$ is the result of a prior choice made by some other participants $q$ and $r$.
Treating dependencies as explicit messages is key to ensuring the distributed implementability of protocols that usual multiparty theories cannot support (e.g.,~$\gtc{G_{\sff{a}}}$~\eqref{eq:mpstMon:auth}).
Our algorithm for monitor synthesis mechanically computes these dependencies from relative types, and exploits them to coordinate monitored blackboxes.

A central ingredient in our technical developments is the notion of \emph{satisfaction} (\Cref{d:mpstMon:satisfaction}), which defines when a monitored blackbox conforms to the role of a specific participant.
Building upon satisfaction, we prove \emph{soundness} and \emph{transparency} for networks of monitored blackboxes.
Soundness (\Cref{t:mpstMon:soundnessMain}) ensures that if each monitored blackbox in a network behaves correctly (according to a global type), then the entire network behaves correctly too.
Transparency (\Cref{t:mpstMon:transparencyMain}) ensures that monitors do not interfere with the (observable) behavior of their contained blackboxes; it is given in terms of a (weak) behavioral equivalence, which is suitably informed by the actions of a given global type.

\paragraph{Related work.}

The literature on distributed runtime verification is vast.
In this setting, the survey by Francalanza \etal~\cite{book/FrancalanzaPS18} proposes several classification criteria.
Phrased in terms of their criteria, our work concerns distributed monitoring for  asynchronous message-passing.
We work with blackboxes, whose monitors are minimally intrusive: they do not alter behavior, but do contribute to coordination.

The works by Bocchi \etal~\cite{journal/tcs/BocchiCDHY17,conf/forte/BocchiCDHY13,conf/tgc/ChenBDHY11} and by Scalas and Yoshida~\cite{conf/popl/ScalasY19}, mentioned above, are a main source of inspiration to us.
The work~\cite{conf/popl/ScalasY19} highlights the limitations of techniques based on the projection of a global type onto local types: many practical protocols, such as $\gtc{G_{\sff{a}}}$, cannot be analyzed because their projection onto local types is undefined.
With respect to~\cite{journal/tcs/BocchiCDHY17,conf/forte/BocchiCDHY13,conf/tgc/ChenBDHY11}, there are three major differences.
First, Bocchi \etal rely on precise specifications of components ($\pi$-calculus processes), whereas we monitor blackboxes (\LTSs).
Second, we resort to relative types, whereas they rely on {local types}; this is a limitation, as just mentioned.
Third, their monitors drop incorrect messages (cf.~\cite{conf/concur/AcetoCFI18}) instead of signaling errors, as we do.
Their framework ensures transparency (akin to \Cref{t:mpstMon:transparencyMain}) and safety, i.e., monitored components do not misbehave.
In contrast, we establish soundness, which is different and more technically involved than safety: our focus is on monitoring blackboxes rather than fully specified components, and soundness concerns correct behavior rather than the absence of misbehavior.

We mention runtime verification techniques based on \emph{binary} session types, a sub-class of multiparty session types.
Bartolo Burlò \etal~\cite{conf/ecoop/BartoloBurloFS21} monitor sequential processes that communicate synchronously, prove that ill-typed processes raise errors, and consider also probabilistic session types~\cite{journal/scico/BartoloBurloFSTT22,conf/coordination/BartoloBurloFSTT21}.
Other works couple monitoring session types with \emph{blame assignment} upon protocol violations~\cite{conf/popl/JiaGP16,journal/jlamp/GommerstadtJP22,conf/tgc/Thiemann14,journal/jlamp/IgarashiTTVW19}.
Jia \etal~\cite{conf/popl/JiaGP16} monitor asynchronous session-typed processes.
Gommerstadt \etal~\cite{journal/jlamp/GommerstadtJP22,conf/esop/GommerstadtJP18} extend~\cite{conf/popl/JiaGP16} with rich refinement-based contracts.
\ifmpstMon
Thiemann~\cite{conf/tgc/Thiemann14} and Igarashi \etal~\cite{journal/jlamp/IgarashiTTVW19,conf/icfp/IgarashiTVW17} use gradual session types to combine statically and dynamically typed languages.
\fi
We do not consider blame assignment, but it can conceivably be added by enhancing error signals.

\paragraph{Outline.}

\Cref{s:mpstMon:networks} defines networks of monitored blackboxes and their behavior.
\Cref{s:mpstMon:synth} defines how to synthesize monitors from global types.
\Cref{s:mpstMon:properties} defines correct monitored blackboxes, and establishes soundness and transparency.
\Cref{s:mpstMon:concl} concludes the \ifthesis{}chapter\else{}paper\fi.
We use colors to improve readability.

\ifappendix
\ifmpstMon \Cref{ac:mpstMon} \else The appendix (\Cpageref{a:exampleBocchi}) \fi includes
\else
\cite{thesis/vdHeuvel24ex} includes an appendix with
\fi
additional examples (including the running example from~\cite{journal/tcs/BocchiCDHY17}), a description of a practical toolkit based on this \ifthesis{}chapter\else{}paper\fi, and omitted proofs.

\section{Networks of Monitored Blackboxes}
\label{s:networks}

We write $P,Q,\ldots$ to denote \emph{blackbox processes} (simply \emph{blackboxes}) that implement protocol \emph{participants} (denoted $p,q,\ldots$).
We assume that a blackbox $P$ is associated with an \LTS that specifies its behavior.
Transitions are denoted~$P \ltrans{\alpha} P'$.
Actions $\alpha$, defined in \Cref{f:mpstMon:netGrammars}~(top), encompass messages $m$, which can be labeled data but also \emph{dependency messages} (simply \emph{dependencies}).
As we will see, dependencies are useful to ensure the coordinated implementation of choices.
Messages abstract away from values, and include only their type.

\begin{figure}[p]
    Let $\ell$ and $T$ denote a label and a data type, respectively.
    \begin{align*}
        & \text{Actions}
        & \alpha,\beta &
        \begin{array}[t]{@{}l@{\kern.5ex}lr@{\kern2ex}l@{\kern.5ex}lr@{\kern2ex}l@{\kern.5ex}lr@{}}
            {} :\deq {}&
            \tau & \text{(silent)}
            & \sepr &
            m & \text{(message)}
            & \sepr &
            \tEnd & \text{(end)}
        \end{array}
        \\
        & \text{Messages}
        & m,n &
        \begin{array}[t]{@{}l@{\kern.5ex}lr@{\kern2ex}l@{\kern.5ex}lr@{}}
            {} :\deq {}&
            p {!} q (\msg \ell<T>) & \text{(output)}
            & \sepr &
            p {!} q \Parens{\ell} & \text{(dep. output, only for networks)}
            \\ \sepr* &
            p {?} q (\msg \ell<T>) & \text{(input)}
            & \sepr &
            p {?} q \Parens{\ell} & \text{(dep. input)}
        \end{array}
    \end{align*}
    Given $\alpha$, the \emph{recipient in $\alpha$} is defined as follows:
    \begin{align*}
        \recip(p {!} q(\msg \ell<T>)) &\deq \recip(p {!} q\Parens{\ell}) \deq q
        \\
        \recip(p {?} q(\msg \ell<T>)) &\deq \recip(p {?} q\Parens{\ell}) \deq \recip(\tau) \deq \recip(\tEnd) \deq \text{undefined}
    \end{align*}

    \hrulefill

    \smallskip
    Let $D,E,\ldots$ denote sets of participants; $I,J,\ldots$ denote non-empty sets of labels; $X,Y,\ldots$ denote recursion variables.
    \begin{align*}
        & \text{Networks}
        & \net P,\net Q &
        \begin{array}[t]{@{}l@{\kern.5ex}lr@{}}
            {} :\deq {} &
            \monImpl{ \bufImpl{ p }{ P }{ \vect m } }{ \mc{M} }{ \vect n } & \text{(monitored blackbox)}
            \\ \sepr* &
            \net P \| \net Q & \text{(parallel composition)}
            \\ \sepr* &
            \pError_D & \text{(error signal)}
        \end{array}
        \\
        & \text{Monitors}
        & \mc{M} &
        \begin{array}[t]{@{}l@{\kern.5ex}lr@{\kern2ex}l@{\kern.5ex}lr@{}}
            {} :\deq {} &
            p {!} q \mBraces{ \msg i<T_i> \mc. \mc{M} }_{i \in I} & \text{(output)}
            & \sepr &
            p {!} D (\ell) \mc. \mc{M} & \text{(dep.\ output)}
            \\ \sepr* &
            p {?} q \mBraces{ \msg i<T_i> \mc. \mc{M} }_{i \in I} & \text{(input)}
            & \sepr &
            p {?} q \mBraces{ i \mc. \mc{M} }_{i \in I} & \text{(dep.\ input)}
            \\ \sepr* &
            \mu X \mc. \mc{M} \sepr X & \text{(recursion)}
            & \sepr &
            \tEnd & \text{(end)}
            \\ \sepr* &
            \pError & \text{(error)}
            & \sepr &
            \checkmark & \text{(finished)}
        \end{array}
    \end{align*}
    The set of \emph{subjects of a network} is as follows:
    \begin{align*}
        \subjs( \monImpl{ \bufImpl{ p }{ P }{ \vect m } }{ \mc{M} }{ \vect n } ) &\deq \braces{ p }
        &
        \subjs(\net P {\|} \net Q) &\deq \subjs(\net P) \cup \subjs(\net Q)
        &
        \subjs(\pError_D) &\deq D
    \end{align*}

    \caption{Actions, messages, networks, and monitors.}\label{f:mpstMon:netGrammars}
\end{figure}

A silent transition $\tau$ denotes an internal computation.
Transitions $p {!} q(\msg \ell<T>)$ and $p {?} q(\msg \ell<T>)$ denote the output and input of a message of type $T$ with label $\ell$ between $p$ and $q$, respectively.
If a message carries no data, we write $\msg \ell<>$ (i.e., the data type is empty).
Dependency outputs are used for monitors, defined below.

We adopt minimal assumptions about the behavior of blackboxes:

\begin{definition}[Assumptions]
    \label{d:mpstMon:blackboxLTSassumptions}
    We assume the following about \LTSs of blackboxes:
    \begin{itemize}

        \item\label{i:mpstMon:bTau}
            \textbf{(Finite $\tau$)}
            Sequences of $\tau$-transitions are finite.

        \item\label{i:mpstMon:bInputOutput}
            \textbf{(Input/Output)}
            Never are input- and output-transitions available at the same time.

        \item\label{i:mpstMon:bEnd}
            \textbf{(End)}
            There are never transitions after an $\tEnd$-transition.

    \end{itemize}
\end{definition}

\begin{figure}[p]
    \begin{center}
        \begin{tikzpicture}[node distance=24mm,every path/.style={->},lbl/.style={font=\scriptsize},tight/.style={inner sep=1pt},loose/.style={inner sep=5pt}]
            \node (ci) {$P_c$};
            \node[right=of ci] (cq) {$P_c^{\sff{q}}$};
            \node[right=of cq] (ce) {$P_c^{\sff{e}}$};
            \node[left=of ci] (cl) {$P_c^{\,\sff{l}}$};
            \draw (ci) edge node[lbl,above,tight] {$c {?} s (\msg \sff{quit}<>)$} (cq)
                    (cq) edge node[lbl,above] {$\tEnd$} (ce)
                    (ci) edge node[lbl,above,tight] {$c {?} s (\msg \sff{login}<>)$} (cl)
                    (cl) edge[bend right=35] node[lbl,above,loose] {$c {!} a (\msg \sff{pwd}<\sff{str}>)$} (ci);

            \node[below=10mm of ci] (si) {$P_s$};
            \node[right=of si] (sq) {$P_s^{\sff{q}}$};
            \node[right=of sq] (se) {$P_s^{\sff{e}}$};
            \node[left=of si] (sl) {$P_s^{\,\sff{l}}$};
            \draw   (si) edge node[lbl,above] {$s {!} c (\msg \sff{quit}<>)$} (sq)
                    (sq) edge node[lbl,above] {$\tEnd$} (se)
                    (si) edge node[lbl,above,tight] {$s {!} c (\msg \sff{login}<>)$} (sl)
                    (sl) edge[bend right=35] node[lbl,above,loose] {$s {?} a (\msg \sff{succ}<\sff{bool}>)$} (si);

            \node[below=10mm of si] (ai) {$P_a$};
            \node[right=of ai] (aqs) {$P_a^{\sff{q}_s}$};
            \node[right=of aqs] (aqc) {$P_a^{\sff{q}_c}$};
            \node[below=8mm of aqc] (ae) {$P_a^{\sff{e}}$};
            \node[left=of ai] (als) {$P_a^{\,\sff{l}_s}$};
            \node[below=8mm of als] (alc) {$P_a^{\,\sff{l}_c}$};
            \node[below=8.7mm of ai] (ap) {$P_a^{\sff{p}}$};
            \draw   (ai)  edge node[lbl,above,tight]{$a {?} s \Parens{\sff{quit}}$} (aqs)
                    (aqs) edge node[lbl,above,tight] {$a {?} c \Parens{\sff{quit}}$} (aqc)
                    (aqc) edge node[lbl,left] {$\tEnd$} (ae)
                    (ai)  edge node[lbl,above,tight] {$a {?} s \Parens{\sff{login}}$} (als)
                    (als)  edge node[lbl,left,tight] {$a {?} c \Parens{\sff{login}}$} (alc)
                    (alc)  edge node[lbl,above] {$a {?} c (\msg \sff{pwd}<\sff{str}>)$} (ap)
                    (ap)  edge node[lbl,right,tight] {$a {!} s (\msg \sff{succ}<\sff{bool}>)$} (ai);
        \end{tikzpicture}
    \end{center}
    
    \caption{Example blackboxes (cf.\ \Cref{x:mpstMon:cImpl}).}
    \label{f:mpstMon:x:cImpl}
\end{figure}

\newpage
\begin{example}
    \label{x:mpstMon:cImpl}
    The blackboxes $P_c$, $P_s$, $P_a$ implement $c$, $s$, $a$, respectively, in $\gtc{G_{\sff{a}}}$~\eqref{eq:mpstMon:auth} with the \LTSs in \Cref{f:mpstMon:x:cImpl}.
    All three \LTSs respect the assumptions in \Cref{d:mpstMon:blackboxLTSassumptions}.
    On the other hand, the following \LTS violates all three assumptions; in particular, there are an input- and an output-transition simultaneously enabled at $Q$:
    \[
        \begin{tikzpicture}[node distance=24mm,every path/.style={->},lbl/.style={font=\scriptsize},tight/.style={inner sep=1pt},loose/.style={inner sep=5pt}]
            \node (ci) {$Q$};
            \node[right=of ci] (cq) {$Q^{\sff{q}}$};
            \node[left=of ci] (cp) {$Q^{\sff{p}}$};
            \draw   (ci) edge node[lbl,above,tight] {$c {?} s (\msg \sff{quit}<>)$} (cq)
                    (cq) edge[bend left=35] node[lbl,above] {$\tEnd$} (ci)
                    (ci) edge node[lbl,above,tight] {$c {!} a (\msg \sff{pwd}<\sff{str}>)$} (cp)
                    (cp) edge[loop left] node[lbl,left] {$\tau$} (cp);
        \end{tikzpicture}
        \tag*{\qedhere}
    \]
\end{example}

Blackboxes communicate {asynchronously}, using buffers (denoted $\vect m$): ordered sequences of messages, with the most-recently received message on the left.
The empty buffer is denoted~$\epsi$.
When a blackbox does an input transition, it attempts to read the message from its buffer.
An output transition places the message in the recipient's buffer; to accommodate this, we mark each blackbox with the participant they implement.
The result is a \emph{buffered blackbox}, denoted $\bufImpl{ p }{ P }{ \vect m }$.

By convention, the buffer of $p$ contains output messages with recipient $p$.
We allow the silent reordering of messages with different senders; this way,  e.g., given $q \neq r$, $\vect m , q {!} p (\msg \ell<T>) , r {!} p (\msg \ell'<T'>) , \vect n$ and $\vect m , r {!} p (\msg \ell'<T'>) , q {!} p (\msg \ell<T>) , \vect n$ are the same.

Having defined standalone (buffered) blackboxes, we now define how they interact in \emph{networks}.
We couple each buffered blackbox with a \emph{monitor} $\mc{M}$, which has its own buffer~$\vect n$.
The result is a \emph{monitored blackbox}, denoted $\monImpl{ \bufImpl{ p }{ P }{ \vect m } }{ \mc{M} }{ \vect n }$.

Monitors define finite state machines that accept sequences of incoming and outgoing messages, as stipulated by some protocol.
An \emph{error} occurs when an exchange does not conform to such a protocol.
Additionally, monitors support the dependencies mentioned earlier: when a blackbox exchanges a message, the monitor broadcasts the message's label to other monitored blackboxes such that they can receive the chosen label and react accordingly.

Networks, defined in \Cref{f:mpstMon:netGrammars} (bottom), are  compositions of monitored blackboxes and \emph{error signals}.
An error signal $\pError_D$ replaces a monitored blackbox when its monitor detects an error involving participants in the set $D$.
Indeed, a participant's error will propagate to the other monitored blackboxes in a network.

Output and (dependency) input monitors check outgoing and incoming (dependency) messages, respectively.
Output dependency monitors $p {!} D (\ell) \mc. \mc{M}$ broadcast $\ell$ to the participants in $D$.
Recursive monitors are encoded by recursive definitions ($\mu X \mc. \mc{M}$) and calls ($X$).
The $\tEnd$ monitor waits for the buffered blackbox to end the protocol.
The $\pError$ monitor denotes an inability to process received messages; it will be used when a sender and recipient send different dependency messages.
The finished monitor $\checkmark$ is self-explanatory.

We now define the behavior of monitored blackboxes in networks:

\begin{definition}[\LTS for Networks]
    \label{d:mpstMon:ltsNetworks}
    We define an \emph{\LTS for networks}, denoted $\net P \ltrans{\alpha} \net Q$, by the rules in \ifmpstMon \Cref{f:mpstMon:ltsNetworks1,f:mpstMon:ltsNetworks2} (\Cpageref{f:mpstMon:ltsNetworks1,f:mpstMon:ltsNetworks2}, respectively) \else \Cref{f:mpstMon:ltsNetworks} (\Cpageref{f:mpstMon:ltsNetworks}) \fi with actions $\alpha$ as in \Cref{f:mpstMon:netGrammars}~(top).

    \noindent
    We write $\net P \trans* \net Q$ to denote a sequence of zero or more $\tau$-transitions $\net P \ltrans{\tau} \ldots \ltrans{\tau} \net Q$, and
    we write $\net P \ntrans$ to denote that there do not exist $\alpha,\net Q$ such that $\net P \ltrans{\alpha} \net Q$.
\end{definition}

\begin{figure}[p]
    \begin{align*}
        & \begin{array}{@{}r@{}l@{}}
            & \monImpl{ \bufImpl{ c }{ P_c }{ \epsi } }{ c {?} s \mBraces{ \msg \sff{quit}<> \mc. \tEnd } }{ \epsi }
            \\
            {}\|{}
            & \monImpl{ \bufImpl{ s }{ P_s }{ \epsi } }{ s {!} c \mBraces{ \msg \sff{quit}<> \mc. \tEnd } }{ \epsi }
        \end{array}
        \ltrans{ \tau }
        \begin{array}{@{}r@{}l@{}}
            & \monImpl{ \bufImpl{ c }{ P_c }{ \epsi } }{ c {?} s \mBraces{ \msg \sff{quit}<> \mc. \tEnd } }{ s {!} c (\msg \sff{quit}<>) }
            \\
            {}\|{}
            & \monImpl{ \bufImpl{ s }{ P_s^{\sff{q}} }{ \epsi } }{ \tEnd }{ \epsi }
        \end{array}
        \\
        & \hspace{25em}
        \downarrow \mkern-4mu {\scriptstyle \tau }
        \\
        & \begin{array}{@{}r@{}l@{}}
            & \monImpl{ \bufImpl{ c }{ P_c^{\sff{e}} }{ \epsi } }{ \checkmark }{ \epsi }
            \\
            {}\|{}
            & \monImpl{ \bufImpl{ s }{ P_s^{\sff{e}} }{ \epsi } }{ \checkmark }{ \epsi }
        \end{array}
        \xleftarrow{\tEnd}\xleftarrow{\tEnd}
        \begin{array}{@{}r@{}l@{}}
            & \monImpl{ \bufImpl{ c }{ P_c^{\sff{q}} }{ \epsi } }{ \tEnd }{ \epsi }
            \\
            {}\|{}
            & \monImpl{ \bufImpl{ s }{ P_s^{\sff{q}} }{ \epsi } }{ \tEnd }{ \epsi }
        \end{array}
        \xleftarrow{\tau}
        \begin{array}{@{}r@{}l@{}}
            & \monImpl{ \bufImpl{ c }{ P_c }{ s {!} c (\msg \sff{quit}<>) } }{ \tEnd }{ \epsi }
            \\
            {}\|{}
            & \monImpl{ \bufImpl{ s }{ P_s^{\sff{q}} }{ \epsi } }{ \tEnd }{ \epsi }
        \end{array}
        \\[-0.7ex] \cline{1-2}
        & \begin{array}{@{}r@{}l@{}}
            & \monImpl{ \bufImpl{ c }{ P_c }{ \epsi } }{ c {?} s \mBraces{ \msg \sff{login}<> \mc. \mc{M_c^{\sff{l}}} } }{ s {!} c (\msg \sff{quit}<>) }
            \\
            {}\|{}
            & \monImpl{ \bufImpl{ s }{ P_s^{\sff{q}} }{ \epsi } }{ \tEnd }{ \epsi }
        \end{array}
        \ltrans{\tau}
        \begin{array}{@{}r@{}l@{}}
            & \pError_{\braces{c}}
            \\
            {}\|{}
            & \monImpl{ \bufImpl{ s }{ P_s^{\sff{q}} }{ \epsi } }{ \tEnd }{ \epsi }
        \end{array}
        \ltrans{\tau}
        \pError_{\braces{c,s}}
    \end{align*}
    \caption{
        The \LTS for networks (\Cref{d:mpstMon:ltsNetworks}) at work:
        transitions of correctly/incorrectly communicating monitored blackboxes of participants of $\gtc{G_{\sff{a}}}$~\eqref{eq:mpstMon:auth}.
        Top: $s$ sends to $c$ label $\sff{quit}$, monitor of $c$ reads message, blackbox of $c$ reads message, both components end.
        Bottom: monitor of $c$ expects $\sff{login}$ message but finds $\sff{quit}$ message so signals error, error propagates to $s$.
    }
    \label{f:mpstMon:trans}
\end{figure}

\begin{figure}[p]
    \def\MathparLineskip{\lineskip=2.3mm}
    \def\defaultHypSeparation{\hskip 1mm}
    \begin{mathpar}
        \begin{bussproof}[buf-out]
            \bussAssume{
                \quad\smash[t]{P \ltrans{ p {!} q (\msg \ell<T>) } P'}
            }
            \bussUn{
                \bufImpl{ p }{ P }{ \vect m }
                \mathrel{\raisebox{-2.5pt}[0ex][0ex]{$\ltrans{ p {!} q (\msg \ell<T>) }$}}
                \bufImpl{ p }{ P' }{ \vect m }
            }
        \end{bussproof}
        \and
        \begin{bussproof}[buf-in]
            \bussAssume{
                \smash[t]{P \ltrans{ p {?} q (\msg \ell<T>) } P'}
            }
            \bussUn{
                \bufImpl{ p }{ P }{ \vect m , q {!} p (\msg \ell<T>) }
                \ltrans{ \tau }
                \bufImpl{ p }{ P' }{ \vect m }
            }
        \end{bussproof}
        \and
        \begin{bussproof}[buf-in-dep]
            \bussAssume{
                \quad\smash[t]{P \ltrans{ p {?} q \Parens{\ell} } P'}
            }
            \bussUn{
                \bufImpl{ p }{ P }{ \vect m , q {!} p \Parens{\ell} }
                \ltrans{ \tau }
                \bufImpl{ p }{ P' }{ \vect m }
            }
        \end{bussproof}
        \and
        \begin{bussproof}[buf-tau]
            \bussAssume{
                \quad\smash[t]{P \ltrans{ \tau } P'}
            }
            \bussUn{
                \bufImpl{ p }{ P }{ \vect m }
                \ltrans{ \tau }
                \bufImpl{ p }{ P' }{ \vect m }
            }
        \end{bussproof}
        \and
        \inferLabel{buf-end}~
        \begin{bussproof}(c)
            \bussAssume{
                P
                \ltrans{ \tEnd }
                P'
            }
            \bussUn{
                \bufImpl{ p }{ P }{ \vect m }
                \mathrel{\raisebox{-2pt}[0ex][0ex]{$\ltrans{ \tEnd }$}}
                \bufImpl{ p }{ P' }{ \vect m }
            }
        \end{bussproof}
    \end{mathpar}
    \hrule
    \begin{mathpar}
        \inferLabel{mon-out}~
        \begin{bussproof}(c)
            \bussAssume{
                \bufImpl{ p }{ P }{ \vect m }
                \ltrans{ p {!} q (\msg j<T_j>) }
                \bufImpl{ p }{ P' }{ \vect m }
            }
            \bussAssume{
                \mc{M} = p {!} q \mBraces{ \msg i<T_i> \mc. \mc{M_i} }_{i \in I}
            }
            \bussAssume{
                j \in I
            }
            \bussTern{
                \monImpl{ \bufImpl{ p }{ P }{ \vect m } }{ \mc{M} }{ \vect n }
                \mathrel{\raisebox{-2.5pt}[0ex][0ex]{$\ltrans{ p {!} q (\msg j<T_j>) }$}}
                \monImpl{ \bufImpl{ p }{ P' }{ \vect m } }{ \mc{M_j} }{ \vect n }
            }
        \end{bussproof}
        \and
        \inferLabel{mon-in}~
        \begin{bussproof}(c)
            \bussAssume{
                \mc{M} = p {?} q \mBraces{ y \mc. \mc{M_i} }_{y \in Y}
            }
            \bussAssume{
                x \in Y
            }
            \bussAssume{
                n' \in \{ q {!} p (x) , q {!} p \Parens{x} \}
            }
            \bussTern{
                \monImpl{ \bufImpl{ p }{ P }{ \vect m } }{ \mc{M} }{ \vect n , n' }
                \ltrans{ \tau }
                \monImpl{ \bufImpl{ p }{ P }{ n' , \vect m } }{ \mc{M_j} }{ \vect n }
            }
        \end{bussproof}
        \and
        \inferLabel{mon-tau}~
        \begin{bussproof}(c)
            \bussAssume{
                \bufImpl{ p }{ P }{ \vect m }
                \ltrans{ \tau }
                \bufImpl{ p }{ P' }{ \vect m' }
            }
            \bussAssume{
                \mc{M} \neq \checkmark
            }
            \bussBin{
                \monImpl{ \bufImpl{ p }{ P }{ \vect m } }{ \mc{M} }{ \vect n }
                \ltrans{ \tau }
                \monImpl{ \bufImpl{ p }{ P' }{ \vect m' } }{ \mc{M} }{ \vect n }
            }
        \end{bussproof}
        \\
        \raisebox{2mm}{$\inferLabel{mon-out-dep}~$}
        \begin{bussproof}
            \bussAssume{ }
            \bussUn{
                \monImpl{ \bufImpl{ p }{ P }{ \vect m } }{ p {!} (D \cup \braces{q}) (\ell)  \mc. \mc{M'} }{ \vect n }
                \mathrel{\raisebox{-2.5pt}[0ex][0ex]{$\ltrans{ p {!} q \Parens{\ell} }$}}
                \monImpl{ \bufImpl{ p }{ P }{ \vect m } }{ p {!} D (\ell) \mc. \mc{M'} }{ \vect n }
            }
        \end{bussproof}
        \and
        \raisebox{2mm}{$\inferLabel{mon-out-dep-empty}~$}
        \begin{bussproof}
            \bussAssume{ }
            \bussUn{
                \monImpl{ \bufImpl{ p }{ P }{ \vect m } }{ p {!} \emptyset (\ell) \mc. \mc{M'} }{ \vect n }
                \ltrans{ \tau }
                \monImpl{ \bufImpl{ p }{ P }{ \vect m } }{ \mc{M'} }{ \vect n }
            }
        \end{bussproof}
        \and
        \inferLabel{mon-rec}~
        \begin{bussproof}(c)
            \bussAssume{
                \monImpl{ \bufImpl{ p }{ P }{ \vect m } }{ \mc{M} \braces{ \mu X \mc. \mc{M} / X } }{ \vect n }
                \ltrans{ \alpha }
                \monImpl{ \bufImpl{ p }{ P' }{ \vect m' } }{ \mc{M'} }{ \vect n' }
            }
            \bussUn{
                \monImpl{ \bufImpl{ p }{ P }{ \vect m } }{ \mu X \mc. \mc{M} }{ \vect n }
                \ltrans{ \alpha }
                \monImpl{ \bufImpl{ p }{ P' }{ \vect m' } }{ \mc{M'} }{ \vect n' }
            }
        \end{bussproof}
        \and
        \inferLabel{mon-end}~
        \begin{bussproof}(c)
            \bussAssume{
                \bufImpl{ p }{ P }{ \vect m }
                \ltrans{ \tEnd }
                \bufImpl{ p }{ P' }{ \vect m }
            }
            \bussUn{
                \monImpl{ \bufImpl{ p }{ P }{ \vect m } }{ \tEnd }{ \epsi }
                \mathrel{\raisebox{-2pt}[0ex][0ex]{$\ltrans{ \tEnd }$}}
                \monImpl{ \bufImpl{ p }{ P' }{ \vect m } }{ \checkmark }{ \epsi }
            }
        \end{bussproof}
    \end{mathpar}
\ifmpstMon
    \caption{\LTS for networks (\Cref{d:mpstMon:ltsNetworks}), part 1 of 2.}\label{f:mpstMon:ltsNetworks1}
\end{figure}
\begin{figure}[t]
    \def\MathparLineskip{\lineskip=2.3mm}
    \def\defaultHypSeparation{\hskip 1mm}
\else
    \hrule
\fi
    \begin{mathpar}
        \begin{bussproof}[error-out]
            \bussAssume{
                \bufImpl{ p }{ P }{ \vect m }
                \ltrans{ p {!} q (\msg j<T_j>) }
                \bufImpl{ p }{ P' }{ \vect m }
            }
            \bussAssume{
                \smash[t]{\begin{array}[b]{@{}l@{}}
                        \mc{M} = p {!} r \mBraces{ \msg i<T_i> \mc. \mc{M_i} }_{i \in I}
                    \implies
                    (r \neq q \vee j \notin I)
                    \\
                    \mc{M} \notin \braces{ \mu X \mc. \mc{M'} , p {!} D (\ell) \mc. \mc{M'} }
                \end{array}}
            }
            \bussBin{
                \monImpl{ \bufImpl{ p }{ P }{ \vect m } }{ \mc{M} }{ \vect n }
                \ltrans{ \tau }
                \pError_{\braces{ p }}
            }
        \end{bussproof}
        \and
        \inferLabel{error-end}~
        \begin{bussproof}(c)
            \bussAssume{
                \bufImpl{ p }{ P }{ \vect m } \ltrans{ \tEnd } \bufImpl{ p }{ P' }{ \vect m }
            }
            \bussAssume{
                \vect n = \epsi \implies \mc{M} \notin \braces{ \mu X \mc. \mc{M'} , p {!} D (\ell) \mc. \mc{M'} , \tEnd }
            }
            \bussBin{
                \monImpl{ \bufImpl{ p }{ P }{ \vect m } }{ \mc{M} }{ \vect n } \ltrans{ \tau } \pError_{\braces{p}}
            }
        \end{bussproof}
        \and
        \inferLabel{error-in}~
        \begin{bussproof}(c)
            \bussAssume{
                \mc{M} = p {?} q \mBraces{ y \mc. \mc{M_y} }_{y \in Y}
            }
            \bussAssume{
                x \notin Y
            }
            \bussAssume{
                n' \in \{ q {!} p (x) , q {!} p \Parens{x} \}
            }
            \bussTern{
                \monImpl{ \bufImpl{ p }{ P }{ \vect m } }{ \mc{M} }{ \vect n , n' }
                \ltrans{ \tau }
                \pError_{\braces{ p }}
            }
        \end{bussproof}
        \and
        \begin{bussproof}[error-mon]
            \bussAssume{ }
            \bussUn{
                \monImpl{ \bufImpl{ p }{ P }{ \vect m } }{ \pError }{ \vect n }
                \ltrans{ \tau }
                \pError_{\braces{ p }}
            }
        \end{bussproof}
        \and
        \begin{bussproof}[par-error]
            \bussAssume{ }
            \bussUn{
                \pError_D \| \monImpl{ \bufImpl{ p }{ P }{ \vect m } }{ \mc{M} }{ \vect n }
                \ltrans{ \tau }
                \pError_{D \cup \braces{ p }}
            }
        \end{bussproof}
    \end{mathpar}
    \hrule
    \begin{mathpar}
        \inferLabel{out-mon-buf}~
        \begin{bussproof}(c)
            \bussAssume{
                \net P
                \ltrans{ n' }
                \net P'
            }
            \bussAssume{
                n' \in \{ q {!} p (x) , q {!} p \Parens{x} \}
            }
            \bussBin{
                \net P \| \monImpl{ \bufImpl{ p }{ P }{ \vect m } }{ \mc{M} }{ \vect n }
                \ltrans{ \tau }
                \net P' \| \monImpl{ \bufImpl{ p }{ P }{ \vect m } }{ \mc{M} }{ n' , \vect n }
            }
        \end{bussproof}
        \and
        \inferLabel{par}~
        \begin{bussproof}(c)
            \bussAssume{
                \net P
                \ltrans{ \alpha }
                \net P'
            }
            \bussAssume{
                \recip(\alpha) \notin \subjs(\net Q)
            }
            \bussBin{
                \net P \| \net Q
                \ltrans{ \alpha }
                \net P' \| \net Q
            }
        \end{bussproof}
        \and
        \inferLabel{cong}~
        \begin{bussproof}(c)
            \bussAssume{
                \net P \equiv \net P'
            }
            \bussAssume{
                \net P'
                \ltrans{ \alpha }
                \net Q'
            }
            \bussAssume{
                \net Q' \equiv \net Q
            }
            \bussTern{
                \net P
                \mathrel{\raisebox{-1pt}[0ex][0ex]{$\ltrans{ \alpha }$}}
                \net Q
            }
        \end{bussproof}
    \end{mathpar}
    \caption{\LTS for networks (\Cref{d:mpstMon:ltsNetworks})\ifmpstMon, part 2 of 2\fi.}\label{f:mpstMon:ltsNetworks2}
\end{figure}

\noindent
\ifmpstMon\Cref{f:mpstMon:ltsNetworks1,f:mpstMon:ltsNetworks2} give \else\Cref{f:ltsNetworks} gives \fi four groups of rules, which we briefly discuss.
The Transition group \ruleLabel*{buf} defines the behavior of a buffered blackbox in terms of the behavior of the blackbox it contains; note that input transitions are hidden as $\tau$-transitions.
The Transition group~\ruleLabel*{mon} defines the behavior of a monitored blackbox when the behavior of the enclosed buffered blackbox concurs with the monitor; again, input transitions are hidden as $\tau$-transitions.

When the behavior of the buffered blackbox does not concur with the monitor, the Transition group~\ruleLabel*{error} replaces the monitored blackbox with an error signal.
Transition~\ruleLabel{par-error} propagates error signals to parallel monitored blackboxes.
If a network parallel to the monitored blackbox of $p$ has an outgoing message with recipient $p$, Transition~\ruleLabel{out-mon-buf} places this message in the buffer of the monitored blackbox as a $\tau$-transition.
Transition~\ruleLabel{par} closes transitions under parallel composition, as long as the recipient in the action of the transition ($\recip(\alpha)$) is not a subject of the composed network ($\subjs(\net Q)$, the participants for which monitored blackboxes and error signals appear in $\net Q$).
Transition~\ruleLabel{cong} closes transitions under $\equiv$, which denotes a congruence that defines parallel composition as commutative and associative.

\Cref{f:mpstMon:trans} shows transitions of correct/incorrect communications in networks.

\section{Monitors for Blackboxes Synthesized from Global Types}
\label{s:globalTypesAsMonitors}

In theories of multiparty session types~\cite{conf/popl/HondaYC08,journal/acm/HondaYC16}, \emph{global types} conveniently describe message-passing protocols between sets of participants from a vantage point.
Here we use them as specifications for monitors in networks (\Cref{alg:mpstMon:gtToMon});
for a local view of such global protocols we use \emph{relative types}~\cite{journal/scico/vdHeuvelP22} \ifappendix (cf.\ \Cref{c:relProj})\fi, which describe the interactions and dependencies between \emph{pairs} of participants.
\ifappendix\else The definitions of global types and relative types / projection can be found in \Cref{c:relProj} on \Cpageref{c:relProj}.\fi

\ifappendix
\begin{definition}[Global and Relative Types]
    \label{d:mpstMon:globalRelativeTypes}
    \begin{align*}
        & \text{Global types}
        & \gtc{G},\gtc{G'} &
        \begin{array}[t]{@{}l@{\kern.5ex}lr@{\kern2ex}l@{\kern.5ex}lr@{}}
            {} :\deq {}&
            p {!} q \gtBraces{ \msg i<T_i> \gtc. \gtc{G} }_{i \in I} & \text{(exchange)}
            & \sepr &
            \tEnd & \text{(end)}
            \\ \sepr* &
            \mu X \gtc. \gtc{G} \sepr X & \text{(recursion)}
        \end{array}
        \\
        & \text{Relative types}\!
        & \rtc{R},\rtc{R'} &
        \begin{array}[t]{@{}l@{\kern.5ex}lr@{\kern2ex}l@{\kern.5ex}lr@{}}
            {} :\deq {}&
            p {!} q \rtBraces{ \msg i<T_i> \rtc. \rtc{R} }_{i \in I} & \text{(exchange)}
            & \sepr &
            \tEnd & \text{(end)}
            \\ \sepr* &
            (p {!} r) {!} q \rtBraces{ i \rtc. \rtc{R} }_{i \in I} & \text{(output dep.)}
            & \sepr &
            \mu X \rtc. \rtc{R} \sepr X & \text{(recursion)}
            \\ \sepr* &
            (p {?} r) {!} q \rtBraces{ i \rtc. \rtc{R} }_{i \in I} & \text{(input dep.)}
        \end{array}
    \end{align*}

    \noindent
    We write $\prt(\gtc{G})$ to denote the set of participants involved in exchanges in $\gtc{G}$.
\end{definition}

\noindent
The global type $p {!} q \gtBraces{ \msg i<T_i> \gtc. \gtc{G_i} }_{i \in I}$ specifies that $p$ sends to $q$ some $j \in I$ with $T_j$, continuing as $\gtc{G_j}$.
A relative type specifies a protocol between a pair of participants, say $p$ and $q$.
The type $p {!} q \rtBraces{ \msg i<T_i> \rtc. \rtc{R_i} }_{i \in I}$ specifies that $p$ sends to $q$ some $j \in I$ with $T_j$, continuing as $\rtc{R_j}$.
If the protocol between $p$ and $q$ depends on a prior choice involving $p$ or $q$, their relative type includes a \emph{dependency}: $(p {!} r) {!} q \rtBraces{ i \rtc. \rtc{R_i} }_{i \in I}$ (resp.\ $(p {?} r) {!} q \rtBraces{ i \rtc. \rtc{R_i} }_{i \in I}$) specifies that $p$ forwards to $q$ the $j \in I$ sent to (resp.\ received from) $r$ by $p$.
For both global and relative types, tail-recursion is defined with recursive definitions $\mu X$ and recursive calls $X$, and $\tEnd$ specifies the end of the protocol.

Relative types are obtained from global types by means of projection:

\begin{definition}[Relative Projection]
    \label{d:mpstMon:relativeProjection}
    The \emph{relative projection} of a global type onto a pair of participants, denoted $\gtc{G} \wrt (p,q)$, is defined by \Cref{alg:mpstMon:relativeProjection}.
\end{definition}

\begin{algorithm}[t!]
    \DontPrintSemicolon
    \SetAlgoNoEnd
    \SetInd{.2em}{.7em}
    \SetSideCommentRight
    \Def{$\gtc{G} \wrt (p,q)$}{

        \Switch{$\gtc{G}$}{

            \uCase{$s {!} r \gtBraces{ \msg i<T_i> \gtc. \gtc{G_i} }_{i \in I}$}{
                $\forall i \in I.~ \rtc{R_i} \deq \gtc{G_i} \wrt (p,q)$ \;

                \lIf{$(p = s \wedge q = r)$}{
                    \KwRet
                    $p {!} q \rtBraces{ \msg i<T_i> \rtc. \rtc{R_i} }_{i \in I}$
                }\label{li:mpstMon:rpEx1}

                \lElseIf{$(q = s \wedge p = r)$}{
                    \KwRet
                    $q {!} p \rtBraces{ \msg i<T_i> \rtc. \rtc{R_i} }_{i \in I}$
                }\label{li:mpstMon:rpEx2}

                \lElseIf{$\forall i,j \in I.~ \rtc{R_i} = \rtc{R_j}$}{
                    \KwRet
                    $\bigcup_{i \in I} \rtc{R_i}$
                }\label{li:mpstMon:rpIndep}

                \lElseIf{$s \in \braces{p,q} \wedge t \in \braces{p,q} \setminus \braces{s}$}{
                    \KwRet
                    $(s {!} r) {!} t \rtBraces{ i \rtc. \rtc{R_i} }_{i \in I}$
                }\label{li:mpstMon:rpDep1}

                \lElseIf{$r \in \braces{p,q} \wedge t \in \braces{p,q} \setminus \braces{r}$}{
                    \KwRet
                    $(r {?} s) {!} t \rtBraces{ i \rtc. \rtc{R_i} }_{i \in I}$
                }\label{li:mpstMon:rpDep2}
            }

            \uCase{$\mu X \gtc. \gtc{G'}$}{
                $\rtc{R'} \deq \gtc{G'} \wrt (p,q)$ \;

                \lIf{$(\text{$\rtc{R'}$ contains an exchange or a recursive call on any $Y \neq X$})$}{
                    \KwRet
                    $\mu X \rtc. \rtc{R'}$
                }\label{li:mpstMon:rpRec}

                \lElse{
                    \KwRet
                    $\tEnd$
                }\label{li:mpstMon:rpRecEnd}
            }

            \lCase{$X$}{
                \KwRet
                $X$
            }\label{li:mpstMon:rpCall}

            \lCase{$\tEnd$}{
                \KwRet
                $\tEnd$
            }\label{li:mpstMon:rpEnd}
        }

    }

    \caption{Relative projeciton of $\gtc{G}$ onto $p$ and $q$ (\Cref{d:mpstMon:relativeProjection}).}
    \label{alg:mpstMon:relativeProjection}
\end{algorithm}

\noindent
The projection of an exchange onto $(p,q)$ is an exchange if $p$ and $q$ are sender and recipient (\cref{li:mpstMon:rpEx1,li:mpstMon:rpEx2}).
Otherwise, if the protocol between $p$ and $q$ does not depend on the exchange (the projections of all branches are equal), the projection is the union of the projected branches (\cref{li:mpstMon:rpIndep}).
The union of relative types, denoted $\rtc{R} \cup \rtc{R'}$, is defined only on identical relative types (e.g., $p {!} q \rtBraces{ \msg i<T_i> \rtc. \rtc{R_i} }_{i \in I} \cup p {!} q \rtBraces{ \msg i<T_i> \rtc. \rtc{R_i} }_{i \in I} = p {!} q \rtBraces{ \msg i<T_i> \rtc. \rtc{R_i} }_{i \in I}$; see \ifappendix \cref{as:mpstMon:relativeTypesWithLocs} \else \ifthesis \cite{thesis/vdHeuvel24ex} \else \cite{report/vdHeuvelPD23} \fi\fi for a formal definition).
If there is a dependency and $p$ or $q$ is sender/recipient, the projection is a dependency (\cref{li:mpstMon:rpDep1,li:mpstMon:rpDep2}).
Projection is undefined if there is a dependency but $p$ nor $q$ is involved.

The projection of $\mu X \gtc. \gtc{G'}$ is a relative type starting with a recursive definition, provided that the projection of $\gtc{G'}$ onto $(p,q)$ contains an exchange or nested recursive call (\cref{li:mpstMon:rpRec}) to avoid recursion with only dependencies;
otherwise, the projection returns $\tEnd$ (\cref{li:mpstMon:rpRecEnd}).
The projections of recursive calls and $\tEnd$ are homomorphic (\cref{li:mpstMon:rpCall,li:mpstMon:rpEnd}).
\fi

\begin{example}
    \label{x:mpstMon:cRelProj}
    The relative projections of $\gtc{G_{\sff{a}}}$~\eqref{eq:mpstMon:auth} are:
    \begin{align*}
        \rtc{R_{c,s}} &\deq \gtc{G_{\sff{a}}} \wrt (c,s) = \mu X \rtc. s {!} c \rtBraces{ \msg \sff{login}<> \rtc. X \rtc, \msg \sff{quit}<> \rtc. \tEnd }
        \\
        \rtc{R_{c,a}} &\deq \gtc{G_{\sff{a}}} \wrt (c,a) = \mu X \rtc. (c {?} s) {!} a \rtBraces{ \sff{login} \rtc. c {!} a \rtBraces{ \msg \sff{pwd}<\sff{str}> \rtc. X } \rtc, \sff{quit} \rtc. \tEnd }
        \\
        \rtc{R_{s,a}} &\deq \gtc{G_{\sff{a}}} \wrt (s,a) = \mu X \rtc. (s {!} c) {!} a \rtBraces{ \sff{login} \rtc. a {!} s \rtBraces{ \msg \sff{succ}<\sff{bool}> \rtc. X } \rtc, \sff{quit} \rtc. \tEnd }
    \end{align*}
    Hence, the exchange from $s$ to $c$ is a dependency for the protocols of $a$.
\end{example}

\ifappendix
Not all global types are sensible.
A valid global type may, e.g., require a participant $p$ to have different behaviors, depending on a choice that $p$ is unaware of (see, e.g.,~\cite{journal/lmcs/CastagnaDP12}).
In the following, we work only with \emph{well-formed} global types:

\begin{definition}[Well-formedness]
    \label{d:mpstMon:wf}
    We say a global type $\gtc{G}$ is \emph{well-formed} if and only if,
    for all pairs of participants $p \neq q \in \prt(\gtc{G})$, the projection $\gtc{G} \wrt (p,q)$ is defined,
    and
    all recursion in $\gtc{G}$ is non-contractive (e.g., $\gtc{G} \neq \mu X \gtc. X$) and bound.
\end{definition}

\noindent
Our running example $\gtc{G_{\sff{a}}}$~\eqref{eq:mpstMon:auth} is well-formed in the above sense; also, as explained in~\cite{conf/popl/ScalasY19}, $\gtc{G_{\sff{a}}}$ is \emph{not} well-formed in most theories of multiparty sessions (based on projection onto local types).
As such, $\gtc{G_{\sff{a}}}$ goes beyond the scope of such theories.

\paragraph{Synthesizing monitors.}
\fi

\ifappendix Next, we \else We \fi define a procedure to synthesize monitors for the participants of global types.
This procedure detects dependencies as follows:

\begin{definition}[Dependence]
    \label{d:mpstMon:depsOn}
    Given a global type $\gtc{G}$, we say $p$ \emph{depends on} $q$ in $\gtc{G}$, denoted $\depsOn p q \gtc{G}$, if and only if
    \\
    \[
        \gtc{G} = s {!} r \gtBraces{ \msg i<T_i> \gtc. \gtc{G_i} }_{i \in I} \wedge p \notin \braces{s,r} \wedge q \in \braces{s,r}
        \wedge \exists i,j \in I.~ \gtc{G_i} \wrt (p,q) \neq \gtc{G_j} \wrt (p,q).
    \]
\end{definition}

\noindent
Thus, $\depsOn p q \gtc{G}$ holds if and only if $\gtc{G}$ is an exchange involving $q$ but not $p$, and the relative projections of at least two branches of the exchange are different.

\begin{definition}[Synthesis of Monitors from Global Types]
    \label{d:mpstMon:gtToMon}
    \Cref{alg:mpstMon:gtToMon} synthesizes the monitor for $p$ in $\gtc{G}$ with participants $D$, denoted $\gtToMon( \gtc{G} , p , D )$.
\end{definition}

\begin{algorithm}[t]
    \DontPrintSemicolon
    \SetAlgoNoEnd
    \SetInd{.2em}{.6em}
    \SetSideCommentRight

    \Def{$\gtToMon(\gtc{G},p,D)$}{
        \Switch{$\gtc{G}$}{

            \uCase{$s {!} r \gtBraces{ \msg i<T_i> \gtc. \gtc{G_i} }_{i \in I}$}{
                $\depsVar \deq \braces{ q \in D \mid \depsOn q p \gtc{G} }$

                \lIf{$p = s$}{
                    \KwRet
                    $p {!} r \mBraces{ \msg i<T_i> \mc. p {!} \depsVar (i) \mc. \gtToMon(\gtc{G_i},p,D) }_{i \in I}$
                }\label{li:mpstMon:monEx1}

                \lElseIf{$p = r$}{
                    \KwRet
                    $p {?} s \mBraces{ \msg i<T_i> \mc. p {!} \depsVar (i) \mc. \gtToMon(\gtc{G_i},p,D) }_{i \in I}$
                }\label{li:mpstMon:monEx2}

                \uElseIf{$p \notin \braces{r,s}$}{
                    $\depOnVar_s \deq (s \in D \wedge \depsOn p s \gtc{G})$ \;
                    $\depOnVar_r \deq (r \in D \wedge \depsOn p r \gtc{G})$

                    \lIf{$(\depOnVar_s \wedge \neg \depOnVar_r)$}{
                        \KwRet
                        $p {?} s \mBraces{ i \mc. \gtToMon(\gtc{G_i},p,D) }_{i \in I}$
                    }\label{li:mpstMon:monDep1}

                    \lElseIf{$(\depOnVar_r \wedge \neg \depOnVar_s)$}{
                        \KwRet
                        $p {?} r \mBraces{ i \mc. \gtToMon(\gtc{G_i},p,D) }_{i \in I}$
                    }\label{li:mpstMon:monDep2}

                    \uElseIf{$(\depOnVar_s \wedge \depOnVar_r)$}{
                        \KwRet
                        $p {?} s \mBraces[\big]{ i \mc. p {?} r \mBraces{ i \mc. \gtToMon(\gtc{G_i},p,D) } \cup \mBraces{ j \mc. \pError }_{j \in I \setminus \braces{i}} }_{i \in I}$
                        \label{li:mpstMon:monDep3}
                    }

                    \lElse({(arbitrary $k \in I$)}){
                        \KwRet
                        $\gtToMon(\gtc{G_k},p,D)$
                    }\label{li:mpstMon:monNoDep}
                }
            }

            \uCase{$\mu X \gtc. \gtc{G'}$}{
                $D' \deq \braces{ q \in D \mid \gtc{G} \wrt (p,q) \neq \tEnd }$

                \lIf{$D' \neq \emptyset$}{
                    \KwRet
                    $\mu X \mc. \gtToMon(\gtc{G'},p,D')$
                }\label{li:mpstMon:monRec}

                \lElse{
                    \KwRet
                    $\tEnd$
                }\label{li:mpstMon:monRecEnd}
            }

            \lCase{$X$}{
                \KwRet
                $X$
            }\label{li:mpstMon:monCall}

            \lCase{$\tEnd$}{
                \KwRet
                $\tEnd$
            }\label{li:mpstMon:monEnd}

        }
    }

    \caption{Synthesis of monitors from global types (\Cref{d:mpstMon:gtToMon}).}
    \label{alg:mpstMon:gtToMon}
\end{algorithm}

\noindent
Initially, $D = \prt(\gtc{G}) \setminus \braces{p}$.
The monitor for $p$ of an exchange where $p$ is sender (resp.\ recipient) is an output (resp.\ input) followed in each branch by a dependency output, using dependence (\Cref{d:mpstMon:depsOn}) to compute the participants with dependencies (\cref{li:mpstMon:monEx1,li:mpstMon:monEx2}).
If $p$ is not involved, we detect a dependency for $p$ with \Cref{d:mpstMon:depsOn}.
In case $p$ depends on sender/recipient but not both, the monitor is a dependency input (\cref{li:mpstMon:monDep1,li:mpstMon:monDep2}).
If $p$ depends on sender \emph{and} recipient, the monitor contains two consecutive dependency inputs (\cref{li:mpstMon:monDep3}); when the two received labels differ, the monitor enters an $\pError$-state.
When there is no dependency for $p$, the monitor uses an arbitrary branch (\cref{li:mpstMon:monNoDep}).
To synthesize a monitor for $\mu X \gtc. \gtc{G'}$, the algorithm uses projection to compute $D'$ with participants having exchanges with $p$ in $\gtc{G'}$ (cf.\ \ifappendix \Cref{alg:mpstMon:relativeProjection} \cref{li:mpstMon:rpRec}\else \cref{li:relProj:rpRec} of \Cref{alg:relProj:relProj} on \Cpageref{alg:relProj:relProj}\fi).
If $D'$ is non-empty, the monitor starts with a recursive definition (\cref{li:mpstMon:monRec}) and the algorithm continues with $D'$; otherwise, the monitor is~$\tEnd$~(\cref{li:mpstMon:monRecEnd}).
The monitors of $X$ and $\tEnd$ are homomorphic (\cref{li:mpstMon:monCall,li:mpstMon:monEnd}).

\begin{example}
    Let us use
    $\gtc{G} \deq p {!} q \gtBraces{ \msg \ell<T> \gtc. \mu X \gtc. p {!} r \gtBraces{ \msg \ell'<T'> \gtc. X \gtc, \msg \ell''<T''> \gtc. \tEnd } }$
    to illustrate \cref{alg:mpstMon:gtToMon}.
    We have $\gtc{G} \wrt (p,q) = p {!} q \rtBraces{ \msg \ell<T> \rtc. \tEnd }$: the projection of the recursive body in $\gtc{G}$ is $(p {!} r) {!} q \rtBraces{ \msg \ell'<T'> \rtc. X \rtc, \msg \ell''<T''> \rtc. \tEnd }$, but there are no exchanges between $p$ and $q$, so the projection of the recursive definition is $\tEnd$.
    Were the monitor for $p$ synthesized with $q \in D$, \Cref{d:mpstMon:depsOn} would detect a dependency: the recursive definition's monitor would be $p {!} r \mBraces{ \msg \ell'<T'> \mc. p {!} \braces{q} (\ell') \mc. X \mc, \msg \ell''<T''> \mc. p {!} \braces{q} (\ell'') \mc. \tEnd }$.
    However, per $\gtc{G} \wrt (p,q)$, $p$ nor $q$ expects a dependency at this point of the protocol.
    Hence, the algorithm removes $q$ from $D$ when entering the recursive body in $\gtc{G}$.
\end{example}

\begin{example}
    \label{x:mpstMon:cMon}
    The monitors of $c$, $s$, $a$ in $\gtc{G_{\sff{a}}}$~\eqref{eq:mpstMon:auth} are:
    \begin{align*}
        \mc{M_c} &\deq \gtToMon(\gtc{G_{\sff{a}}},c,\braces{s,a})
        \\
        &=
        \mu X \mc. c {?} s \mBraces*{
            \begin{array}{@{}l@{}}
                \msg \sff{login}<> \mc. c {!} \braces{a} (\sff{login}) \mc. c {!} a \mBraces{ \msg \sff{pwd}<\sff{str}> \mc. c {!} \emptyset (\sff{pwd}) \mc. X } \mc,
                \\
                \msg \sff{quit}<> \mc. c {!} \braces{a} (\sff{quit}) \mc. \tEnd
            \end{array}
        }
        \\
        \mc{M_s} &\deq \gtToMon(\gtc{G_{\sff{a}}},s,\braces{c,a})
        \\
        &=
        \mu X \mc. s {!} c \mBraces*{
            \begin{array}{@{}l@{}}
                \msg \sff{login}<> \mc. s {!} \braces{a} (\sff{login}) \mc. s {?} a \mBraces{ \msg \sff{succ}<\sff{bool}> \mc. s {!} \emptyset (\sff{succ}) \mc. X } \mc,
                \\
                \msg \sff{quit}<> \mc. s {!} \braces{a} (\sff{quit}) \mc. \tEnd
            \end{array}
        }
        \\
        \mc{M_a} &\deq \gtToMon(\gtc{G_{\sff{a}}},a,\braces{c,s})
        \\
        &=
        \mu X \mc. a {?} s \mBraces*{
            \begin{array}{@{}l@{}}
                \sff{login} \mc. a {?} c \mBraces*{
                    \begin{array}{@{}l@{}}
                        \sff{login} \mc. \begin{array}[t]{@{}l@{}}
                            a {?} c
                            \color{mClr}\lbrace\mskip-5mu\lbrace\normalcolor
                            \msg \sff{pwd}<\sff{str}> \mc. a {!} \emptyset (\sff{pwd}) \mc.
                            \\
                            a {!} s \mBraces{ \msg \sff{succ}<\sff{bool}> \mc. a {!} \emptyset (\sff{succ}) \mc. X }
                            \color{mClr}\rbrace\mskip-5mu\rbrace\normalcolor
                            \mc,
                        \end{array}
                        \\
                        \sff{quit} \mc. \pError
                    \end{array}
                } \mc,
                \\
                \sff{quit} \mc. a {?} c \mBraces{ \sff{quit} \mc. \tEnd \mc, \sff{login} \mc. \pError }
            \end{array}
        }
        \ifmpstMon \tag*{\qedhere} \fi
    \end{align*}
\end{example}

\section{Properties of Correct Monitored Blackboxes}
\label{s:correctMonitoredBlackboxes}

Given a global type $\gtc{G}$, we establish the precise conditions under which a network of monitored blackboxes correctly implements $\gtc{G}$.
That is, we define how the monitored blackbox $\net P$ of a participant $p$ of $\gtc{G}$ should behave, i.e., when $\net P$ \emph{satisfies} the role of $p$ in $\gtc{G}$ (satisfaction, \Cref{d:mpstMon:satisfaction}).
We then prove two important properties of networks of monitored blackboxes that satisfy a given global type:
\begin{itemize}
    \item[\textbf{Soundness:}]
        The network behaves correctly according to the global type (\Cref{t:mpstMon:soundnessMain});

    \item[\textbf{Transparency:}]
        The monitors interfere minimally with buffered blackboxes (\Cref{t:mpstMon:transparencyMain}).
\end{itemize}

\noindent
As we will see in \Cref{s:mpstMon:properties:soundness}, satisfaction is exactly the condition under which a network $\net P$ is sound with respect to a global type $\gtc{G}$.

\subsection{Satisfaction}
\label{s:satisfaction}

Our aim is to attest that $\net P$ satisfies the role of $p$ in $\gtc{G}$ if it meets certain conditions on the behavior of monitored blackboxes with respect to the protocol.
As we have seen, the role of~$p$ in $\gtc{G}$ is determined by  projection.
Satisfaction is then a relation $\mcl{R}$ between (i)~monitored blackboxes and (ii)~maps from participants $q \in \prt(\gtc{G}) \setminus \braces{p}$ to relative types between $p$ and~$q$, denoted $\RTs$; $\mcl{R}$ must contain $(\net P,\RTs)$ with relative projections of $\gtc{G}$.
Given any $(\net P',\RTs')$ in $\mcl{R}$, the general idea of satisfaction is (i)~that an output to $q$ by $\net P'$ means that $\RTs'(q)$ is a corresponding exchange from $p$ to $q$, and (ii)~that if there is a $q$ such that $\RTs'(q)$ is an exchange from $q$ to $p$ then $\net P'$ behaves correctly afterwards.

In satisfaction, dependencies in relative types require care.
For example, if $\RTs'(q)$ is an exchange from $p$ to $q$ and $\RTs'(r)$ is a dependency on this exchange, then $\net P'$ must first send a label to $q$ and then send \emph{the same label} to $r$.
Hence, we need to track the labels chosen by the monitored blackbox for later reference.
To this end, we uniquely identify each exchange in a global type by its \emph{location} $\vect \ell$: a sequence of labels denoting the choices leading to the exchange.
Projection then uses these locations to annotate each exchange and recursive definition/call in the relative type it produces.
Because projection skips independent exchanges (\ifappendix\Cref{alg:mpstMon:relativeProjection}, \cref{li:mpstMon:rpIndep}\else{}cf.\ ``$\ddep$'' in \Cref{d:relProj:relProj} on \Cpageref{d:relProj:relProj}\fi), some exchanges and recursive definitions/calls may be associated with multiple locations; hence, they are annotated with \emph{sets} of locations, denoted~$\mbb L$.
Satisfaction then tracks choices using a map from sets of locations to labels, denoted~$\Lbls$.
Projection with location annotations is formally defined in \ifappendix \Cref{as:mpstMon:relativeTypesWithLocs}\else\ifthesis \cite{thesis/vdHeuvel24ex}\else \cite{report/vdHeuvelPD23}\fi\fi, along with a corresponding definition for unfolding recursion.

Before defining satisfaction, we  set  up some useful notation for type signatures, locations, relative types, and maps.

\begin{notation}
    Let $\bm{P}$ denote the set of all participants, $\bm{R}$ the set of all relative types, $\bm{N}$ the set of all networks, and $\bm{L}$ the set of all labels.

    Notation $\Pow(S)$ denotes the powerset of $S$.
    Given a set $S$, we write $\vect S$ to denote the set of all sequences of elements from $S$.
    We write $\mbb L \overlap \mbb L'$ to stand for $\mbb L \cap \mbb L' \neq \emptyset$.
    We write $\mbb L \leq \mbb L'$ if every $\vect \ell' \in \mbb L'$ is prefixed by some $\vect \ell \in \mbb L$.

    In relative types, we write $\lozenge$ to denote either ${!}$ or ${?}$.
    We write $\unfold(\rtc{R})$ for the inductive unfolding of $\rtc{R}$ if $\rtc{R}$ starts with recursive definitions, and for $\rtc{R}$ itself otherwise.
    We write $\rtc{R} \unfoldeq \rtc{R'}$ whenever $\unfold(\rtc{R}) = \unfold(\rtc{R'})$.

    We shall use monospaced fonts to denote maps (such as $\RTs$ and $\Lbls$).
    We often define maps using the notation of injective relations.
    Given a map $\mathtt{M}$, we write $(x,y) \in \mathtt{M}$ to denote that~\mbox{$x \in \dom(\mathtt{M})$} and $\mathtt{M}(x) = y$.
    We write $\mathtt{M} \update{x \mapsto y'}$ to denote the map obtained by adding to $\mathtt{M}$ an entry for $x$ pointing to $y'$, or updating an already existing entry for $x$.
    Maps are partial.
\end{notation}

\begin{figure}[t]
    Given $( \net P , \RTs , \Lbls ) \in \mcl{R}$, all the following conditions hold:
    \begin{enumerate}[itemsep=2.0pt]
        \item\label{i:mpstMon:sTau}
            \textbf{(Tau)}
            If $\net P \ltrans{ \tau } \net P'$,
            then $( \net P' , \RTs , \Lbls ) \in \mcl{R}$.

        \item\label{i:mpstMon:sEnd}
            \textbf{(End)}
            If $\net P \ltrans{ \tEnd } \net P'$,
            then, for every $(q,\rtc{R}) \in \RTs$,
            $\rtc{R} \unfoldeq \tEnd$,
            $\net P' \ntrans$, and $( \net P' , \emptyset , \emptyset ) \in \mcl{R}$.

        \item\label{i:mpstMon:sOutput}
            \textbf{(Output)}
            If $\net P \ltrans{ p {!} q (\msg j<T_j>) } \net P'$,
            then
            $\RTs(q)\unfoldeq p {!} q^{\mbb L} \rtBraces{ \msg i<T_i> \rtc. \rtc{R_i} }_{i \in I}$
            with $j \in I$, and $( \net P' , \RTs \update{q \mapsto \rtc{R_j}} , \Lbls \update{\mbb L \mapsto j} ) \in \mcl{R}$.

        \item\label{i:mpstMon:sInput}
            \textbf{(Input)}
            If there is $(q,\rtc{R}) \in \RTs$ such that
            $\rtc{R} \unfoldeq q {!} p^{\mbb L} \rtBraces{ \msg i<T_i> \rtc. \rtc{R_i} }_{i \in I}$,
            then
            \mbox{$\net P = \monImpl{ \bufImpl{ p }{ P }{ \vect m } }{ \mc{M} }{ \vect n }$},
            and, for every $j \in I$,
            \\
            \mbox{$( \monImpl{ \bufImpl{ p }{ P }{ \vect m } }{ \mc{M} }{ q {!} p (\msg j<T_j>) , \vect n } , \RTs \update{q \mapsto \rtc{R_j}} , \Lbls \update{\mbb L \mapsto j} ) \in \mcl{R}$}.

        \item\label{i:mpstMon:sDependencyOutput}
            \textbf{(Dependency output\kern-1pt)}
            If $\net P \mkern-3mu \ltrans{ p {!} q \Parens{j} } \mkern-3mu \net P'$,
            then
            $\RTs(q) \unfoldeq (p \lozenge r) {!} q^{\mbb L} \rtBraces{ i \rtc. \rtc{R_i} }_{i \in I}$
            with $j \in I$,
            there is $(\mbb L',j) \in \Lbls$ such that $\mbb L' \overlap \mbb L$,
            and $( \net P' , \RTs \update{q \mapsto \rtc{R_j}} , \Lbls ) \in \mcl{R}$.

        \item\label{i:mpstMon:sDependencyInput}
            \textbf{(Dependency input)}
            If there is $(q,\rtc{R}) \in \RTs$ s.t.\
            $\rtc{R} \unfoldeq (q \lozenge r) {!} p^{\mbb L} \rtBraces{ i \rtc. \rtc{R_i} }_{i \in I}$,
            then \mbox{$\net P = \monImpl{ \bufImpl{ p }{ P }{ \vect m } }{ \mc{M} }{ \vect n }$}, and either of the following holds:
            \begin{itemize}
                \item
                    (Fresh label)
                    there is no $\mbb L' \in \dom(\Lbls)$ such that $\mbb L' \overlap \mbb L$, and, for every $j \in I$, $( \monImpl{ \bufImpl{ p }{ P }{ \vect m } }{ \mc{M} }{ q {!} p \Parens{j} , \vect n } , \RTs \update{q \mapsto \rtc{R_j}} , \Lbls \update{\mbb L \mapsto j} ) \in \mcl{R}$;

                \item
                    (Known label)
                    there is $(\mbb L',j) \in \Lbls$ such that $\mbb L' \overlap \mbb L$ and $j \in I$, and
                    \\
                    \mbox{$( \monImpl{ \bufImpl{ p }{ P }{ \vect m } }{ \mc{M} }{ q {!} p \Parens{j} , \vect n } , \RTs \update{q \mapsto \rtc{R_j}} , \Lbls ) \in \mcl{R}$}.
            \end{itemize}
    \end{enumerate}
    \caption{Satisfaction: conditions under which $\mcl{R}$ holds at $p$ (\Cref{d:mpstMon:satisfaction}).}
    \label{f:mpstMon:satisfaction}
\end{figure}

\begin{definition}[Satisfaction]
    \label{d:mpstMon:satisfaction}
    A relation $\mcl{R}$ is \emph{sat-signed} if its signature is $\bm{N} \times (\bm{P} \rightarrow \bm{R}) \times (\Pow(\vect {\bm{L}}) \rightarrow \bm{L})$.
    We define the following properties of relations:
    \begin{itemize}
        \item
            A sat-signed relation $\mcl{R}$ \emph{holds at $p$} if it satisfies the conditions in \Cref{f:mpstMon:satisfaction}.

        \item
            A sat-signed relation $\mcl{R}$  \emph{progresses at $p$} if for every $( \net P , \RTs , \Lbls ) \in \mcl{R}$, we have $\net P \ltrans{\alpha} \net P'$ for some $\alpha$ and $\net P'$, given that one of the following holds:
            \begin{itemize}
                \item
                    $\RTs \neq \emptyset$ and, for every $(q,\rtc{R}) \in \RTs$, $\rtc{R} \unfoldeq \tEnd$;

                \item
                    There is $(q,\rtc{R}) \in \RTs$ such that (i)~$\rtc{R} \unfoldeq p {!} q^{\mbb L} \rtBraces{ \msg i<T_i> \rtc. \rtc{R_i} }_{i \in I}$ or $\rtc{R} \unfoldeq (p \lozenge r) {!} q^{\mbb L} \rtBraces{ i \rtc. \rtc{R_i} }_{i \in I}$, and (ii)~for every $(q',\rtc{R'}) \in \RTs \setminus \braces{(q,\rtc{R})}$, either
                    $\rtc{R'} \unfoldeq \tEnd$
                    or $\unfold(\rtc{R'})$ has locations~$\mbb L'$ with $\mbb L \leq \mbb L'$.
            \end{itemize}

        \item
            A sat-signed relation $\mcl{R}$ is a \emph{satisfaction at $p$} if it holds and progresses at $p$.
    \end{itemize}
    We write $\mcl{R} \satisfies{ \net P }[ \Lbls ]{ \RTs }[ p ]$ if $\mcl{R}$ is a satisfaction at $p$ with $(\net P , \RTs , \Lbls) \in \mcl{R}$, and
    \mbox{$\mcl{R} \satisfies{ \net P }{ \RTs }[ p ]$} when $\Lbls$ is empty.
    \ifmpstMon
    We write $\satisfies{ \net P }[ \Lbls ]{ \RTs }[ p ]$ if there exists $\mcl{R}$ such that $\mcl{R} \satisfies{ \net P }[ \Lbls ]{ \RTs }[ p ]$.
    \else
    We omit $\mcl{R}$ to indicate such $\mcl{R}$ exists.
    \fi
\end{definition}

\noindent
Satisfaction requires $\mcl{R}$ to hold at $p$: each $(\net P,\RTs,\Lbls) \in \mcl{R}$ enjoys the conditions in \Cref{f:mpstMon:satisfaction}, discussed next, which ensure that $\net P$ respects the protocols in $\RTs$.

\satclause{sTau}{Tau} allows $\tau$-transitions without affecting $\RTs$ and $\Lbls$.
\satclause{sEnd}{End} allows an $\tEnd$-transition, given that all relative types in $\RTs$ are $\tEnd$.
The resulting state should not transition, enforced by empty $\RTs$ and $\Lbls$.

\satclause{sOutput}{Output} allows an output-transition with a message to $q$, given that $\RTs(q)$ is a corresponding output by $p$.
Then, $\RTs(q)$ updates to the continuation of the appropriate branch and $\Lbls$ records the choice under the locations of $\RTs(q)$.

\satclause{sInput}{Input} triggers when there is $(q,\rtc{R}) \in \RTs$ where $\rtc{R}$ is a message from $q$ to~$p$.
Satisfaction targets the standalone behavior of $\net P$, so we simulate a message sent by $q$.
The resulting behavior is analyzed after buffering any such message; $\RTs(q)$ is updated to the continuation of the corresponding branch.
Again, $\Lbls$ records the choice at the locations of~$\RTs(q)$.

\satclause{sDependencyOutput}{Dependency Output} allows an output-transition with a dependency message to $q$, given that $\RTs(q)$ is a corresponding dependency output by $p$ with locations~$\mbb L$.
The message's label should be recorded in $\Lbls$ at some $\mbb L'$ that shares a location with $\mbb L$: here $\mbb L'$ relates to a past exchange between $p$ and some $r$ in $\gtc{G}$ from which the dependency output in $\RTs(q)$ originates.
This ensures that the dependency output is preceded by a corresponding exchange, and that the dependency output carries the same label as originally chosen for the preceding exchange.
Afterwards, $\RTs(q)$ is updated to the continuation of the appropriate branch.

\satclause{sDependencyInput}{Dependency Input} triggers when there is $(q,\rtc{R}) \in \RTs$ such that $\rtc{R}$ is a dependency exchange from $q$ to $p$, forwarding a label exchanged between $q$ and~$r$.
As in the input case, a message from $q$ is simulated by buffering it in $\net P$.
In this case, $\RTs(r)$ could be a dependency exchange from $r$ to $p$, originating from the same exchange between $q$ and $r$ in $\gtc{G}$.
To ensure that the buffered messages contain the same label, we distinguish ``fresh'' and ``known'' cases.
In the fresh case, we consider the first of the possibly two dependency exchanges: there is no $\mbb L' \in \dom(\Lbls)$ that shares a location with the locations $\mbb L$ of $\RTs(q)$.
Hence, we analyze each possible dependency message, updating $\RTs(q)$ appropriately and recording the choice in~$\Lbls$.
The known case then considers the second dependency exchange: there is a label in~$\Lbls$ at $\mbb L'$ that shares a location with $\mbb L$.
Hence, we buffer a message with the same label, and update $\RTs(q)$ accordingly.

Satisfaction also requires $\mcl{R}$ to progress at $p$, for each $(\net P,\RTs,\Lbls) \in \mcl{R}$ making sure that~$\net P$ does not idle whenever we are expecting a  transition from~$\net P$.
There are two cases.
(1)~If all relative types in $\RTs$ are $\tEnd$, we expect an $\tEnd$-transition.
(2)~If there is a relative type in $\RTs$ that is a (dependency) output,
we expect an output transition.
However, $\net P$ may idle if it is waiting for a message: there is $(q,\rtc{R}) \in \RTs$ such that $\rtc{R}$ is a (dependency) input originating from an exchange in $\gtc{G}$ that precedes the exchange related to the output.

\begin{definition}[Satisfaction for Networks]
    \label{d:mpstMon:monSat}
    Let us write $\RTsOf(\gtc{G},p)$ to denote the set \mbox{$\braces{ ( q , \gtc{G} \wrt (p,q) ) \mid q \in \prt(\gtc{G}) \setminus \braces{p} }$}.
    Moreover, we write
    \begin{itemize}

        \item
            $\mcl{R} \satisfies{ \monImpl{ \bufImpl{ p }{ P }{ \epsi } }{ \mc{M} }{ \epsi } }{ \gtc{G} }[ p ]$ if and only if $\mc{M} = \gtToMon( \gtc{G} , p , \prt(\gtc{G}) \setminus \braces{p} )$ and \mbox{$\mcl{R} \satisfies{ \monImpl{ \bufImpl{ p }{ P }{ \epsi } }{ \mc{M} }{ \epsi } }{ \RTsOf(\gtc{G},p) }[ p ]$}.
            \ifmpstMon
            We write $\satisfies{ \monImpl{ \bufImpl{ p }{ P }{ \epsi } }{ \mc{M} }{ \epsi } }{ \gtc{G} }[ p ]$ to denote that there exists $\mcl{R}$ such that $\mcl{R} \satisfies{ \monImpl{ \bufImpl{ p }{ P }{ \epsi } }{ \mc{M} }{ \epsi } }{ \gtc{G} }[ p ]$.
            \else
            We omit $\mcl{R}$ to say such $\mcl{R}$ exists.
            \fi

        \item
            $\satisfies{ \net P }{ \gtc{G} }$ if and only if $\net P \equiv \prod_{p \in \prt(\gtc{G})} \monImpl{ \bufImpl{ p }{ P_p }{ \epsi } }{ \mc{M_p} }{ \epsi }$ and, for every $p \in \prt(\gtc{G})$, \mbox{$\satisfies{ \monImpl{ \bufImpl{ p }{ P_p }{ \epsi } }{ \mc{M_p} }{ \epsi } }{ \gtc{G} }[ p ]$}.

    \end{itemize}
\end{definition}

\begin{example}
    The following satisfaction assertions hold with implementations, relative types, and monitors from \Cref{x:mpstMon:cImpl,x:mpstMon:cRelProj,x:mpstMon:cMon}, respectively:
    \begin{align*}
        & \satisfies{ \monImpl{ \bufImpl{ c }{ P_c }{ \epsi } }{ \mc{M_c} }{ \epsi } }{ \braces{ ( s , \rtc{R_{c,s}} ) , ( a , \rtc{R_{c,a}} ) } }[ c ]
        \\
        & \satisfies{ \monImpl{ \bufImpl{ s }{ P_s }{ \epsi } }{ \mc{M_s} }{ \epsi } }{ \braces{ ( c , \rtc{R_{c,s}} ) , ( a , \rtc{R_{s,a}} ) } }[ s ]
        \\
        & \satisfies{ \monImpl{ \bufImpl{ a }{ P_a }{ \epsi } }{ \mc{M_a} }{ \epsi } }{ \braces{ ( c , \rtc{R_{c,a}} ) , ( s , \rtc{R_{s,a}} ) } }[ a ]
        \\
        & \satisfies{ \monImpl{ \bufImpl{ c }{ P_c }{ \epsi } }{ \mc{M_c} }{ \epsi } \| \monImpl{ \bufImpl{ s }{ P_s }{ \epsi } }{ \mc{M_s} }{ \epsi } \| \monImpl{ \bufImpl{ a }{ P_a }{ \epsi } }{ \mc{M_a} }{ \epsi } }{ \gtc{G_{\sff{a}}} }
    \end{align*}
    We also have:
    $\satisfies[\nvDash]{ \monImpl{ \bufImpl{ c }{ P_c }{ \epsi } }{ \mu X \mc. c {?} s \mBraces{ \msg \sff{quit}<> \mc. \tEnd } }{ \epsi } }{ \gtc{G_{\sff{a}}} }[ c ]$. This is because
    $\rtc{R_{c,s}}$ specifies that $s$ may send $\sff{login}$ to $c$, which this monitor would not accept.
\end{example}

\subsection{Soundness}
\label{s:soundness}

Our \emph{first result} is that satisfaction is sound with respect to global types: when a network of monitored blackboxes satisfies a global type $\gtc{G}$ (\Cref{d:mpstMon:monSat}), any path of transitions eventually reaches a state that satisfies another global type reachable from $\gtc{G}$.
Hence, the satisfaction of the individual components that a network comprises is enough to ensure that the network behaves as specified by the global type.
Reachability between global types is defined as an \LTS:

\begin{definition}[\LTS for Global Types]
    \label{d:mpstMon:ltsGlobalTypes}
    We define an \LTS for global types, denoted $\gtc{G} \ltrans{ \ell } \gtc{G'}$, by the following rules:
    \begin{mathpar}
        \begin{bussproof}[glob-exch]
            \bussAssume{
                j \in I
            }
            \bussUn{
                p {!} q \gtBraces{ \msg i<T_i> \gtc. \gtc{G_i} }_{i \in I}
                \mathrel{\raisebox{-1.5pt}[0ex][0ex]{$\ltrans{ j }$}}
                \gtc{G_j}
            }
        \end{bussproof}
        \and
        \begin{bussproof}[glob-rec]
            \bussAssume{
                \gtc{G} \braces{ \mu X \gtc. \gtc{G} / X }
                \ltrans{ \ell }
                \gtc{G'}
            }
            \bussUn{
                \mu X \gtc. \gtc{G}
                \mathrel{\raisebox{-1.5pt}[0ex][0ex]{$\ltrans{ \ell }$}}
                \gtc{G'}
            }
        \end{bussproof}
    \end{mathpar}

    \noindent
    Given $\vect \ell = \ell_1 , \ldots , \ell_n$, we write $\gtc{G} \ltrans{ \vect \ell } \gtc{G'}$ to denote $\gtc{G} \ltrans{ \ell_1 } \ldots \ltrans{ \ell_n } \gtc{G'}$.
\end{definition}

\begin{theorem}[Soundness]
    \label{t:mpstMon:soundnessMain}
    If $\satisfies{\net P}{\gtc{G}}$ (\Cref{d:mpstMon:monSat}) and $\net P \trans* \net P_0$ then there exist~$\gtc{G'},\vect \ell,\net P'$ such that $\gtc{G} \ltrans{\vect \ell} \gtc{G'}$, $\net P_0 \trans* \net P'$, and $\satisfies{\net P'}{\gtc{G'}}$.
\end{theorem}

We sketch the proof of \Cref{t:mpstMon:soundnessMain} (see \ifappendix \Cref{as:mpstMon:soundnessProof} \else\ifthesis \cite{thesis/vdHeuvel24ex}\else \cite{report/vdHeuvelPD23}\fi\fi~for details).
We prove a stronger statement that starts from a network $\net P$ that satisfies an intermediate $\gtc{G_0}$ reachable from $\gtc{G}$.
This way, we apply induction on the number of transitions between $\net P$ and $\net P_0$, relating the transitions of the network to the transitions of $\gtc{G_0}$ one step at a time by relying on satisfaction (\Cref{d:mpstMon:satisfaction}).
Hence, we inductively ``consume'' the transitions between $\net P$ and $\net P_0$ until we have passed through $\net P_0$ and end up in a network satisfying $\gtc{G'}$ reachable from~$\gtc{G_0}$.
We use an auxiliary lemma to account for global types with \emph{independent exchanges}, such as \mbox{$\gtc{G'} = p {!} q \gtBraces{ \msg \ell<T> \gtc. r {!} s \gtBraces{ \msg \ell'<T'> \gtc. \tEnd } }$}.
In $\gtc{G'}$, the exchange involving $(p,q)$ is unrelated to that involving~$(r,s)$, so they occur concurrently in a network implementing $\gtc{G'}$.
Hence, the transitions from~$\net P$ to $\net P_0$ might not follow the order specified in $\gtc{G_0}$.
The lemma ensures that concurrent (i.e., unrelated) transitions always end up in the same state, no matter the order.
This way we show transitions from $\net P$ in the order specified in $\gtc{G_0}$, which we restore to the observed order using the lemma when we are done.

\Cref{t:mpstMon:soundnessMain} implies that any  $\net P$ that satisfies some global type is \emph{error free}, i.e., $\net P$ never reduces to a network containing $\pError_D$ (\ifappendix\Cref{t:mpstMon:errorFreedom}, \Cref{as:mpstMon:soundnessProof}\else{}stated and proved formally in\ifthesis~\cite{thesis/vdHeuvel24ex}\else the full version of this paper~\cite{report/vdHeuvelPD23}\fi\fi).

\subsection{Transparency}
\label{s:transparency}

The task of monitors is to observe and verify behavior \emph{with minimal interference}: monitors should be \emph{transparent}.
Transparency is usually expressed as a bisimulation between a monitored and unmonitored component~\cite{journal/ijis/LigattiBW05,journal/ijsttt/FalconeFM12,journal/tcs/BocchiCDHY17,conf/concur/AcetoCFI18}.

Our \emph{second result} is thus a transparency result.
For it to be informative, we assume that we observe the (un)monitored blackbox as if it were running in a network of monitored blackboxes that adhere to a given global protocol.
This way, we can assume that received messages are correct, such that the monitor does not transition to an error signal.
To this end, we enhance the \LTS for networks (\Cref{d:mpstMon:ltsNetworks}):
\begin{enumerate}
    \item
        As in satisfaction (\Cref{d:mpstMon:satisfaction}), we consider (un)monitored blackboxes on their own.
        Hence, we need a way to simulate messages sent by other participants.
        Otherwise, a blackbox would get stuck waiting for a message and the bisimulation would hold trivially.
        We thus add a transition that buffers messages.
        Similar to \satref{sInput}{Input} and \satclause{sDependencyInput}{Dependency Input}, these messages cannot be arbitrary; we parameterize the enhanced \LTS by an \emph{oracle} that determines which messages are allowed as stipulated by a given global type.

    \item
        Besides observing and verifying transitions, our monitors additionally send dependency messages.
        This leads to an asymmetry in the behavior of monitored blackboxes and unmonitered blackboxes, as the latter do not send dependency messages.
        Hence, we rename dependency output actions to $\tau$.
\end{enumerate}

\noindent
We now define the enhanced \LTS for networks, after setting up some notation.

\begin{notation}
    Let $\bm{A}$ denote the set of all actions.
    Given $\Omega : \Pow(\vect {\bm{A}})$, we write $\alpha + \Omega$ to denote the set containing every sequence in $\Omega$ prepended with $\alpha$.
    We write $\Omega(\alpha) = \Omega'$ iff $\alpha + \Omega' \subseteq \Omega$ and there is no $\Omega''$ such that $\alpha + \Omega' \subset \alpha + \Omega'' \subseteq \Omega$.
\end{notation}

\begin{figure}[t]
    \begin{mathpar}
        \def\defaultHypSeparation{\hskip .8ex}
        \inferLabel{buf-mon}~
        \begin{bussproof}(c)
            \bussAssume{
                \alpha = p {?} q (x), n' = q {!} p (x)
                \mathbin{\text{or}}
                \alpha = p {?} q \Parens{x}, n' = q {!} p \Parens{x}
            }
            \bussAssume{
                \Omega(\alpha) = \Omega'
            }
            \bussBin{
                \raisebox{-8.0pt}[0ex][0ex]{$
                    \monImpl{ \bufImpl{ p }{ P }{ \vect m } }{ {\mc{M}} }{ \vect n }
                    \mathrel{\raisebox{-3.0pt}[0ex][0ex]{$\bLtrans{\Omega}{ \alpha }{\Omega'}$}}
                    \monImpl{ \bufImpl{ p }{ P }{ \vect m } }{ \mc{M} }{ n' , \vect n }
                $}
            }
        \end{bussproof}
        \and
        \inferLabel{buf-unmon}~
        \begin{bussproof}(c)
            \bussAssume{
                \vphantom{\raisebox{7pt}{$\net P$}}
                \alpha = p {?} q (x), m' = q {!} p (x)
                \mathbin{\text{or}}
                \alpha = p {?} q \Parens{x}, m' = q {!} p \Parens{x}
            }
            \bussAssume{
                \Omega(\alpha) = \Omega'
            }
            \bussBin{
                \raisebox{-8.0pt}[0ex][0ex]{$
                    \bufImpl{ p }{ P }{ \vect m }
                    \mathrel{\raisebox{-3.0pt}[0ex][0ex]{$\bLtrans{\Omega}{ \alpha }{\Omega'}$}}
                    \bufImpl{ p }{ P }{ m' , \vect m }
                $}
            }
        \end{bussproof}
        \\
        \inferLabel{dep}~
        \begin{bussproof}(c)
            \bussAssume{
                \vphantom{\raisebox{7pt}{$\net P$}}
                \net P
                \ltrans{ p {!} q \Parens{\ell} }
                \net P'
            }
            \bussUn{
                \net P
                \bLtrans{\Omega}{ \tau }{\Omega}
                \net P'
            }
        \end{bussproof}
        \and
        \inferLabel{no-dep}~
        \begin{bussproof}(c)
            \bussAssume{
                \vphantom{\raisebox{7pt}{$\net P$}}
                \net P
                \ltrans{ \alpha }
                \net P'
            }
            \bussAssume{
                \alpha \notin \braces{ p {?} q (x), p {?} q \Parens{x} , p {!} q \Parens{\ell} }
            }
            \bussAssume{
                \Omega(\alpha) = \Omega'
            }
            \bussTern{
                \net P
                \bLtrans{\Omega}{ \alpha }{\Omega'}
                \net P'
            }
        \end{bussproof}
    \end{mathpar}
    \caption{Enhanced \LTS for networks (\Cref{d:mpstMon:enhancedLTS}).}\label{f:mpstMon:enhancedLTS}
\end{figure}

\begin{definition}[Enhanced \LTS for Networks]
    \label{d:mpstMon:enhancedLTS}
    We define an \emph{enhanced \LTS for networks}, denoted $\net P \bLtrans{\Omega}{\alpha}{\Omega'} \net P'$ where $\Omega , \Omega' : \Pow(\vect {\bm{A}})$, by the rules in \Cref{f:mpstMon:enhancedLTS}.
    \ifmpstMon

    \else\fi
    We write $\net P \bLtrans*{\Omega}{~~}{\Omega} \net P'$
    whenever $\net P$ transitions to $\net P'$ in zero or more $\tau$-transitions, i.e., $\net P \bLtrans{\Omega}{ \tau }{\Omega} \cdots \bLtrans{\Omega}{ \tau }{\Omega} \net P'$.
    We write $\net P \bLtrans*{\Omega}{ \alpha }{\Omega'} \net P'$ when $\net P \bLtrans*{\Omega}{~~}{\Omega} \net P_1 \bLtrans{\Omega}{\alpha}{\Omega'} \net P_2 \bLtrans*{\Omega'}{~~}{\Omega'} \net P'$, omitting the $\alpha$-transition when $\alpha = \tau$.
    Given $\vect \alpha = \alpha_1 , \ldots , \alpha_n$, we write $\net P \bLtrans*{\Omega_0}{ \vect \alpha }{\Omega_n} \net P'$ when $\net P \bLtrans*{\Omega_0}{\alpha_1}{\Omega_1} \net P_1 \cdots \net P_{n-1} \bLtrans*{\Omega_{n-1}}{\alpha_n}{\Omega_n} \net P'$.
\end{definition}

\noindent
Thus, Transitions~\ruleLabel*{buf} simulate messages from other participants, consulting $\Omega$ and transforming it into $\Omega'$.
Transition~\ruleLabel{dep} renames dependency outputs to $\tau$.
Transition~\ruleLabel{no-dep} passes any other transitions, updating  $\Omega$ to $\Omega'$ accordingly.

We now define a weak bisimilarity on networks, governed by oracles.

\begin{definition}[Bisimilarity]
    \label{d:mpstMon:weakBisim}
    A relation $\mcl{B} : \bm{N} \times \Pow(\vect {\bm{A}}) \times \bm{N}$ is a \emph{(weak) bisimulation} if, for every $( \net P , \Omega , \net Q ) \in \mcl{B}$:
    \ifmpstMon
    \begin{enumerate}
    \fi
        \ifmpstMon\item\else(1)~\fi{}For every $\net P' , \alpha , \Omega_1$ such that $\net P \bLtrans{\Omega}{ \alpha }{\Omega_1} \net P'$, there exist $\vect b , \Omega_2 , \net Q' , \net P''$ such that ${\net Q \bLtrans*{\Omega}{ \vect b , \alpha }{\Omega_2} \net Q'}$, $\net P' \bLtrans*{\Omega_1}{ \vect b }{\Omega_2} \net P''$, and $( \net P'' , \Omega_2 , \net Q' ) \in \mcl{B}$;
        \ifmpstMon\item\else{}and (2)~\fi{}The symmetric analog.
    \ifmpstMon
    \end{enumerate}
    \fi
    We say $\net P$ and $\net Q$ are \emph{bisimilar} with respect to $\Omega$, denoted $\net P \weakBisim{}{\Omega} \net Q$, if there exists a bisimulation~$\mcl{B}$ such that $( \net P , \Omega , \net Q ) \in \mcl{B}$.
\end{definition}

\noindent
Clause~1 says that $\net Q$ can mimic a transition from $\net P$ to $\net P'$, possibly after $\tau$- and \ruleLabel{buf}-transitions.
We then allow $\net P'$ to ``catch up'' on those additional transitions, after which the results are bisimilar (under a new oracle); Clause~2 is symmetric.
Additional \ruleLabel{buf}-transitions are necessary: an unmonitored blackbox can read messages from its buffer directly, whereas a monitor may need to move messages between buffers first.
If the monitor first needs to move messages that are not in its buffer yet, we need to add those messages with \ruleLabel{buf}-transitions.
The unmonitored blackbox then needs to catch up on those additional messages.

Similar to soundness, satisfaction defines the conditions under which we prove  transparency of monitors.
Moreover, we need to define the precise oracle under which  bisimilarity holds.
This oracle is defined similarly to satisfaction: it depends on actions observed, relative types (in $\RTs$), and prior choices (in $\Lbls$).

\begin{definition}[Label Oracle]
    \label{d:mpstMon:labelOracle}
    The \emph{label oracle} of participant $p$ under $\RTs : \bm{P} \rightarrow \bm{R}$ and $\Lbls : \Pow(\vect {\bm{L}}) \rightarrow \bm{L}$, denoted $\LO(p,\RTs,\Lbls)$ is defined in \Cref{f:mpstMon:labelOracle}.
\end{definition}

\begin{figure}[t]
    To improve readability, below we write `$\bigcup [x \in S \ldots]$' instead of `$\bigcup_{x \in S \ldots}$'.
    \begin{align*}
        \hphantom{\cup}
        ~&
        \bigcup [(q,\rtc{R}) \in \RTs .~ \rtc{R} \unfoldeq p {!} q^{\mbb L} \rtBraces{ \msg i<T_i> \rtc. \rtc{R_i} }_{i \in I}]
        ~
        \bigcup [j \in I]
        \\ & \quad
        p {!} q (\msg j<T_j>) + \LO( p , \RTs \update{q \mapsto \rtc{R_j}} , \Lbls \update{\mbb L \mapsto j} )
        \tag*{(Output)}
        \displaybreak[0] \\[.3mm] \cup
        ~&
        \bigcup [(q,\rtc{R}) \in \RTs .~ \rtc{R} \unfoldeq q {!} p^{\mbb L} \rtBraces{ \msg i<T_i> \rtc. \rtc{R_i} }_{i \in I}]
        ~
        \bigcup [j \in I]
        \\ & \quad
        p {?} q (\msg j<T_j>) + \LO( p , \RTs \update{q \mapsto \rtc{R_j}} , \Lbls \update{\mbb L \mapsto j} )
        \tag*{(Input)}
        \displaybreak[0] \\[.3mm] \cup
        ~&
        \bigcup [(q,\rtc{R}) \in \RTs .~ \rtc{R} \unfoldeq (q \lozenge r) {!} p^{\mbb L} \rtBraces{ i \rtc. \rtc{R_i} }_{i \in I} \wedge {\not\exists} \mbb L' \in \dom(\Lbls) .~ \mbb L' \overlap \mbb L]
        ~
        \bigcup [j \in I]
        \\ & \quad
        p {?} q (j) + \LO( p , \RTs \update{q \mapsto \rtc{R_j}} , \Lbls \update{\mbb L \mapsto j} )
        \tag*{(Fresh Dependency Input)}
        \displaybreak[0] \\[.3mm] \cup
        ~&
        \bigcup [(q,\rtc{R}) \in \RTs .~ \rtc{R} \unfoldeq (q \lozenge r) {!} p^{\mbb L} \rtBraces{ i \rtc. \rtc{R_i} }_{i \in I} \wedge \exists (\mbb L',j) \in \Lbls .~ \mbb L' \overlap \mbb L \wedge j \in I]
        \\ & \quad
        p {?} q (j) + \LO( p , \RTs \update{q \mapsto \rtc{R_j}} , \Lbls )
        \tag*{(Known Dependency Input)}
        \displaybreak[0] \\[.3mm] \cup
        ~&
        \bigcup [(q,\rtc{R}) \in \RTs .~ \rtc{R} \unfoldeq (p \lozenge r) {!} q^{\mbb L} \rtBraces{ i \rtc. \rtc{R_i} }_{i \in I} \wedge \exists (\mbb L',j) \in \Lbls .~ \mbb L' \overlap \mbb L \wedge j \in I]
        \\ & \quad
        \LO( p , \RTs \update{q \mapsto \rtc{R_j}} , \Lbls )
        \tag*{(Dependency Output)}
        \displaybreak[0] \\[.3mm] \cup
        ~&
        \tEnd + \LO( p , \emptyset , \emptyset ) \quad [\text{only if } \forall ( q , \rtc{R} ) \in \RTs .~  \rtc{R} \unfoldeq \tEnd]
        \tag*{(End)}
        \displaybreak[0] \\[.3mm] \cup
        ~&
        \tau + \LO( p , \RTs , \Lbls )
        \tag*{(Tau)}
    \end{align*}
    \caption{Definition of the label oracle (\Cref{d:mpstMon:labelOracle}), $\LO(p,\RTs,\Lbls)$.}\label{f:mpstMon:labelOracle}
\end{figure}

\noindent
The label oracle $\LO(p,\RTs,\Lbls)$ thus consists of several subsets, each resembling a condition of satisfaction in \Cref{f:mpstMon:satisfaction}.
Dependency outputs are exempt: the enhanced \LTS renames them to $\tau$, so the label oracle simply looks past them without requiring a dependency output.

We now state our transparency result, after defining a final requirement: minimality of satisfaction.
This allows us to step backward through satisfaction relations, such that we can reason about buffered messages.

\begin{definition}[Minimal Satisfaction]
    \label{d:mpstMon:minimalSat}
    We write $\satisfies[\vdash]{ \net P }{ \gtc{G} }[ p ]$ whenever there exists $\mcl{R}$ such that $\mcl{R} \satisfies{ \net P }{ \RTsOf(\gtc{G},p) }[ p ]$ (\Cref{d:mpstMon:monSat}) and $\mcl{R}$ is \emph{minimal}, i.e., there is no $\mcl{R}' \subset \mcl{R}$ such that $\mcl{R}' \satisfies{ \net P }{ \RTsOf(\gtc{G},p) }[ p ]$.
\end{definition}

\begin{theorem}[Transparency]
    \label{t:mpstMon:transparencyMain}
    Suppose $\satisfies[\vdash]{ \monImpl{ \bufImpl{ p }{ P }{ \epsi } }{ \mc{M} }{ \epsi } }{ \gtc{G} }[ p ]$ (\Cref{d:mpstMon:minimalSat}).
    Let $\Omega \deq \LO(p,\RTsOf(\gtc{G},p),\emptyset)$.
    Then $\monImpl{ \bufImpl{ p }{ P }{ \epsi } }{ \mc{M} }{ \epsi } \weakBisim{}{\Omega} \bufImpl{ p }{ P }{ \epsi }$.
\end{theorem}

We sketch the proof of \Cref{t:mpstMon:transparencyMain} (see \ifappendix \Cref{as:mpstMon:transparencyProof}\else\ifthesis \cite{thesis/vdHeuvel24ex}\else \cite{report/vdHeuvelPD23}\fi\fi).
The minimal satisfaction of the monitored blackbox contains all states that the monitored blackbox can reach through transitions.
We create a relation $\mcl{B}$ by pairing each such state $\monImpl{ \bufImpl{ p }{ P' }{ \vect m } }{ \mc{M'} }{ \vect n }$ with \mbox{$\bufImpl{ p }{ P' }{ \vect n , \vect m }$}---notice how the buffers are combined.
We do so while keeping an informative relation between relative types, monitors, buffers, and oracles.
This information gives us the appropriate oracles to include in $\mcl{B}$.
We then show that $\mcl{B}$ is a weak bisimulation by proving that the initial monitored and unmonitored blackbox are in $\mcl{B}$, and that the conditions of \Cref{d:mpstMon:weakBisim} hold.
While Clause~1 is straightforward, Clause~2 requires care: by using the relation between relative types, monitors, and buffers, we infer the shape of the monitor from a transition of the unmonitored blackbox.
This allows us to show that the monitored blackbox can mimic the transition, possibly after outputting dependencies and/or receiving additional messages (as discussed above).

\smallskip

We close by comparing our \Cref{t:mpstMon:soundnessMain,t:mpstMon:transparencyMain} with Bocchi \etal's safety and transparency results~\cite{journal/tcs/BocchiCDHY17}, respectively.
First, their safety result~\cite[Thm.~5.2]{journal/tcs/BocchiCDHY17} guarantees satisfaction instead of assuming it; their framework suppresses unexpected messages, which prevents the  informative guarantee given by our \Cref{t:mpstMon:soundnessMain}.
Second, \Cref{t:mpstMon:transparencyMain} and their transparency result~\cite[Thm.~6.1]{journal/tcs/BocchiCDHY17} differ, among other things, in the presence of an oracle, which is not needed in their setting: they can inspect the inputs of monitored processes, whereas we cannot verify the inputs of a blackbox without actually sending messages to it.

\section{Conclusion}
\label{s:conclusion}

We have proposed a new framework for dynamically analyzing networks of communicating components (blackboxes), governed by global types, with minimal assumptions about observable behavior.
We use global types and relative projection~\cite{journal/scico/vdHeuvelP22} to synthesize monitors, and define when a monitored component satisfies the governing protocol.
We prove that networks of correct monitored components are sound with respect to a global type, and that monitors are transparent.

We have implemented a practical toolkit, called \texttt{RelaMon}, based on the framework presented here.
\texttt{RelaMon} allows users to deploy JavaScript programs that monitor web-applications in any programming language and with third-party/closed-source components according to a global type.
The toolkit is publicly available~\cite{web/DobreHP23} and
includes implementations of our running example (the global type $\gtc{G_{\sff{a}}}$), as well as an example that incorporates a closed-source weather API.
\ifappendix
\Cref{as:mpstMon:toolkit} includes more details.
\fi

As future work, we plan to extend our framework to uniformly analyze systems combining monitored blackboxes and statically checked components (following~\cite{journal/scico/vdHeuvelP22}).
We also plan to study under which restrictions our approach coincides with Bocchi \etal's~\cite{journal/tcs/BocchiCDHY17}.

\paragraph{Acknowledgments}
We are grateful to
the anonymous reviewers for useful remarks.

\newpage
\addcontentsline{toc}{section}{References}
\bibliography{refs}

\appendices

\newpage
\section{The Running Example from~\cite{journal/tcs/BocchiCDHY17}}
\label{a:exampleBocchi}

Bocchi \etal~\cite{journal/tcs/BocchiCDHY17} develop a running example that is similar to ours.
Their example concerns an ATM protocol between a client ($c$) and a payment server ($s$), preceded by a client authorization through a separate authenticator ($a$).

Bocchi \etal's example includes assertions.
Assertions are orthogonal to the method of projecting global types onto local procotols and extracting monitors from global types; adding assertions does not modify the spirit of our approach.
Next we present $\gtc{G_{\sff{ATM}}}$, a version of Bocchi \etal's running example without assertions; it allows us to illustrate how our approach covers protocols considered in Bocchi \etal's approach.
\begin{align*}
    \gtc{G_{\sff{ATM}}} &\deq c {!} a \gtBraces{ \msg \sff{login}<\sff{str}> \gtc. a {!} s \gtBraces{ \msg \sff{ok}<> \gtc. a {!} c \gtBraces{ \msg \sff{ok}<> \gtc. \gtc{G_{\sff{loop}}} } \gtc, \msg \sff{fail}<> \gtc. a {!} c \gtBraces{ \msg \sff{fail}<> \gtc. \tEnd } } }
    \\
    \gtc{G_{\sff{loop}}} &\deq \mu X \gtc. s {!} c \gtBraces{ \msg \sff{account}<\sff{int}> \gtc. c {!} s \gtBraces{ \msg \sff{withdraw}<\sff{int}> \gtc. X \gtc, \msg \sff{deposit}<\sff{int}> \gtc. X \gtc, \msg \sff{quit}<> \gtc. \tEnd } }
\end{align*}
Notice how, for this example to work under traditional forms of projection, $a$ needs to explicitly forward the success of the login attempt to $c$.

Our framework supports $\gtc{G_{\sff{ATM}}}$ as is, because it is well-formed according to \Cref{d:mpstMon:wf}.
The relative projections attesting to this are as follows (cf.~\Cref{alg:mpstMon:relativeProjection}):
\begin{align*}
    \gtc{G_{\sff{loop}}} \wrt (c,s) &= \mu X \rtc. s {!} c \rtBraces{ \msg \sff{account}<\sff{int}> \rtc. c {!} s \rtBraces{ \msg \sff{withdraw}<\sff{int}> \rtc. X \rtc, \msg \sff{deposit}<\sff{int}> \rtc. X \rtc, \msg \sff{quit}<> \rtc. \tEnd } }
    \\
    \gtc{G_{\sff{ATM}}} \wrt (c,s) &= (s {?} a) {!} c \rtBraces{ \sff{ok} \rtc. \big( \gtc{G_{\sff{loop}}} \wrt (c,s) \big) \rtc, \sff{quit} \rtc. \tEnd }
    \\
    \gtc{G_{\sff{loop}}} \wrt (c,a) &= \tEnd
    \\
    \gtc{G_{\sff{ATM}}} \wrt (c,a) &= c {!} a \rtBraces{ \msg \sff{login}<\sff{str}> \rtc. (a {!} s) {!} c \rtBraces{ \sff{ok} \rtc. a {!} c \rtBraces{ \msg \sff{ok}<> \rtc. \tEnd } \rtc, \sff{quit} \rtc. a {!} c \rtBraces{ \msg \sff{quit}<> \rtc. \tEnd } } }
    \\
    \gtc{G_{\sff{loop}}} \wrt (s,a) &= \tEnd
    \\
    \gtc{G_{\sff{ATM}}} \wrt (s,a) &= a {!} s \rtBraces{ \msg \sff{ok}<> \rtc. \tEnd \rtc, \msg \sff{fail}<> \rtc. \tEnd }
\end{align*}

As mentioned before, $\gtc{G_{\sff{ATM}}}$ contains an explicit dependency.
We can modify the global type to make this dependency implicit, without altering $\gtc{G_{\sff{loop}}}$:
\[
    \gtc{G'_{\sff{ATM}}} \deq c {!} a \gtBraces{ \msg \sff{login}<\sff{str}> \gtc. a {!} s \gtBraces{ \msg \sff{ok}<> \gtc. \gtc{G_{\sff{loop}}} \gtc, \msg \sff{fail}<> \gtc. \tEnd } }
\]
The resulting relative projections are then as follows.
Notice how the change has simplified the projection onto $(c,a)$.
\begin{align*}
    \rtc{R'_{c,s}} \deq \gtc{G'_{\sff{ATM}}} \wrt (c,s) &= (s {?} a) {!} c \rtBraces{ \sff{ok} \rtc. \big(\gtc{G_{\sff{loop}}} \wrt (c,s)\big) \rtc, \sff{quit} \rtc. \tEnd }
    \\
    \rtc{R'_{c,a}} \deq \gtc{G'_{\sff{ATM}}} \wrt (c,a) &= c {!} a \rtBraces{ \msg \sff{login}<\sff{str}> \gtc. \tEnd }
    \\
    \rtc{R'_{s,a}} \deq \gtc{G'_{\sff{ATM}}} \wrt (s,a) &= a {!} s \rtBraces{ \msg \sff{ok}<> \rtc. \tEnd \rtc, \msg \sff{quit}<> \rtc. \tEnd }
\end{align*}

Using \Cref{alg:mpstMon:gtToMon}, we extract monitors from $\gtc{G'_{\sff{ATM}}}$:
\begin{align*}
    &
    \gtToMon(\gtc{G_{\sff{loop}}},c,\braces{s,a})
    \\
    &= \mu X \mc. c {?} s \mBraces{
        \msg \sff{account}<\sff{int}> \mc. c {!} \emptyset (\sff{account}) \mc. c {!} s \mBraces*{
            \begin{array}{@{}l@{}}
                \msg \sff{withdraw}<\sff{int}> \mc. c {!} \emptyset (\sff{withdraw}) \mc. X \mc,
                \\
                \msg \sff{deposit}<\sff{int}> \mc. c {!} \emptyset (\sff{deposit}) \mc. X \mc,
                \\
                \msg \sff{quit}<> \mc. c {!} \emptyset (\sff{quit}) \mc. \tEnd
            \end{array}
        }
    }
    \\
    \mc{M'_c} &\deq \gtToMon(\gtc{G'_{\sff{ATM}}},c,\braces{s,a})
    \\
    &= c {!} a \mBraces{ \msg \sff{login}<\sff{str}> \mc. c {!} \emptyset (\sff{login}) \mc. c {?} s \mBraces{ \sff{ok} \mc. \gtToMon(\gtc{G_{\sff{loop}}},c,\braces{s,a}) \mc, \sff{fail} \mc. \tEnd } }
    \\
    &
    \gtToMon(\gtc{G_{\sff{loop}}},s,\braces{c,a})
    \\
    &= \mu X \mc. s {!} c \mBraces{
        \msg \sff{account}<\sff{int}> \mc. s {!} \emptyset (\sff{account}) \mc. s {?} c \mBraces*{
            \begin{array}{@{}l@{}}
                \msg \sff{withdraw}<\sff{int}> \mc. s {!} \emptyset (\sff{withdraw}) \mc. X \mc,
                \\
                \msg \sff{deposit}<\sff{int}> \mc. s {!} \emptyset (\sff{deposit}) \mc. X \mc,
                \\
                \msg \sff{quit}<> \mc. s {!} \emptyset (\sff{quit}) \mc. \tEnd
            \end{array}
        }
    }
    \\
    \mc{M'_s} &\deq \gtToMon(\gtc{G'_{\sff{ATM}}},s,\braces{c,a})
    \\
    &= s {?} a \mBraces{ \msg \sff{ok}<> \mc. s {!} \braces{c} (\sff{ok}) \mc. \gtToMon(\gtc{G_{\sff{loop}}},s,\braces{c,a}) \mc, \msg \sff{fail}<> \mc. s {!} \braces{c} (\sff{fail}) \mc. \tEnd }
    \\
    \mc{M'_a} &\deq \gtToMon(\gtc{G'_{\sff{ATM}}},a,\braces{c,s})
    \\
    &= a {?} c \mBraces{ \msg \sff{login}<\sff{str}> \mc. a {!} \emptyset (\sff{login}) \mc. a {!} s \mBraces{ \msg \sff{ok}<> \mc. a {!} \emptyset (\sff{ok}) \mc. \tEnd \mc, \msg \sff{fail}<> \mc. a {!} \emptyset (\sff{fail}) \mc. \tEnd } }
\end{align*}

\Cref{f:mpstMon:ATMLTSs} gives example blackboxes for the participants of $\gtc{G'_{\sff{ATM}}}$.
\begin{figure}[t]
    \begin{center}
        \begin{tikzpicture}[node distance=24mm,every path/.style={->},lbl/.style={font=\scriptsize},tight/.style={inner sep=1pt},loose/.style={inner sep=5pt}]
            \node (ci) {$Q_c$};
            \node[right=of ci] (cl) {$Q_c^{\,\sff{l}}$};
            \node[right=of cl] (cf) {$Q_c^{\sff{f}}$};
            \node[right=of cf] (ce) {$Q_c^{\sff{e}}$};
            \node[below=of cl] (co) {$Q_c^{\sff{o}}$};
            \node[right=of co] (ca) {$Q_c^{\sff{a}}$};
            \node[right=of ca] (cq) {$Q_c^{\sff{q}}$};
            \draw (ci) edge node[lbl,above,tight] {$c {!} a (\msg \sff{login}<\sff{str}>)$} (cl)
                    (cl) edge node[lbl,above,tight] {$c {?} s \Parens{\sff{fail}}$} (cf)
                    (cf) edge node[lbl,above] {$\tEnd$} (ce)
                    (cl) edge node[lbl,left,loose] {$c {?} s \Parens{\sff{ok}}$} (co)
                    (co) edge[bend left=35] node[lbl,above,tight] {$c {?} s (\msg \sff{account}<\sff{int}>)$} (ca)
                    (ca) edge[bend left=35] node[lbl,above,loose] {$\begin{array}{@{}r@{}}c {!} s (\msg \sff{withdraw}<\sff{int}>)\\{}c {!} s (\msg \sff{deposit}<\sff{int}>)\end{array}$} (co)
                    (ca) edge node[lbl,above,tight] {$c {!} s (\msg \sff{quit}<>)$} (cq)
                    (cq) edge node[lbl,right] {$\tEnd$} (ce);

              \node[below=50mm of ci] (si) {$Q_s$};
              \node[right=of si] (so) {$Q_s^{\sff{o}}$};
              \node[right=of so] (sa) {$Q_s^{\sff{a}}$};
              \node[below=of sa] (sq) {$Q_s^{\sff{q}}$};
              \node[right=of sq] (se) {$Q_s^{\sff{e}}$};

              \draw (si) edge node[lbl,above,tight] {$s {?} a (\msg \sff{ok}<>)$} (so)
                        (so) edge[bend left=35] node[lbl,above,loose] {$s {!} c (\msg \sff{account}<\sff{int}>)$} (sa)
                        (sa) edge[bend left=35] node[lbl,above,loose] {$\begin{array}{@{}c@{}}s {?} c (\msg \sff{withdraw}<\sff{int}>)\\{}s {?} c (\msg \sff{deposit}<\sff{int}>)\end{array}$} (so)
                        (sa) edge node[lbl,right] {$s {?} c (\msg \sff{quit}<>)$} (sq)
                        (si) edge[bend right=25] node[lbl,right,xshift=5mm] {$s {?} a (\msg \sff{fail}<>)$} (sq)
                        (sq) edge node[lbl,above] {$\tEnd$} (se);

                \node[below=45mm of si] (ai) {$Q_a$};
                \node[right=of ai] (al) {$Q_a^{\,\sff{l}}$};
                \node[right=of al] (aq) {$Q_a^{\sff{q}}$};
                \node[right=of aq] (ae) {$Q_a^{\sff{e}}$};
                \draw (ai) edge node[lbl,above,tight] {$a {?} c (\msg \sff{login}<\sff{str}>)$} (al)
                        (al) edge[bend left=35] node[lbl,above] {$a {!} c (\msg \sff{ok}<>)$} (aq)
                        (al) edge[bend right=35] node[lbl,above] {$a {!} c (\msg \sff{fail}<>)$} (aq)
                        (aq) edge node[lbl,above] {$\tEnd$} (ae);
        \end{tikzpicture}
    \end{center}
    \caption{Blackbox LTSs for the participants of $\gtc{G'_{\sff{ATM}}}$.}
    \label{f:mpstMon:ATMLTSs}
\end{figure}

It is not difficult to confirm that the following satisfactions hold (cf.\ \Cref{d:mpstMon:satisfaction,d:mpstMon:monSat}):
\begin{align*}
    & \satisfies{ \monImpl{ \bufImpl{ c }{ Q_c }{ \epsi } }{ \mc{M'_c} }{ \epsi } }{ \braces{ ( s , \rtc{R'_{c,s}} ) , ( a , \rtc{R'_{c,a}} ) } }[ c ]
    \\
    & \satisfies{ \monImpl{ \bufImpl{ s }{ Q_s }{ \epsi } }{ \mc{M'_s} }{ \epsi } }{ \braces{ ( c , \rtc{R'_{c,s}} ) , ( a , \rtc{R'_{s,a}} ) } }[ s ]
    \\
    & \satisfies{ \monImpl{ \bufImpl{ a }{ Q_a }{ \epsi } }{ \mc{M'_a} }{ \epsi } }{ \braces{ ( c , \rtc{R'_{c,a}} ) , ( s , \rtc{R'_{s,a}} ) } }[ a ]
    \\
    & \satisfies{ \monImpl{ \bufImpl{ c }{ Q_c }{ \epsi } }{ \mc{M'_c} }{ \epsi } \| \monImpl{ \bufImpl{ s }{ Q_s }{ \epsi } }{ \mc{M'_s} }{ \epsi } \| \monImpl{ \bufImpl{ a }{ Q_a }{ \epsi } }{ \mc{M'_a} }{ \epsi } }{ \gtc{G_{\sff{a}}} }
\end{align*}

\newpage
\section{A Toolkit for Monitoring Networks of Blackboxes in Practice}
\label{a:toolkit}

To demonstrate the practical potential of our approach, we have developed a toolkit based on our framework---see \url{https://github.com/basvdheuvel/RelaMon}~\cite{web/DobreHP23}.

The toolkit enhances message-passing web-applications with monitors.
This way, it is possible to add a layer of security when communicating with, e.g., untrusted third-party APIs by monitoring their behavior according to an assumed governing protocol.
The toolkit includes:
\begin{enumerate}
    \item
        A tool, written in Rascal~\cite{conf/scam/KlintSV09,conf/scam/KlintSV19}, that transpiles protocols specified as well-formed global types to JSON.

    \item
        A monitor microservice, written in JavaScript, initialized with a protocol specification, a participant ID, the IP-addresses of the unmonitored component and the other components.
        The microservice uses relative projection on the supplied JSON protocol specification to construct a finite state machine using \Cref{alg:mpstMon:gtToMon}, which acts as the monitor for the specified participant.
\end{enumerate}
When all components and their respective monitors have been deployed, the monitors perform a handshake such that all components are ready to start executing the protocol.
The monitors forward all correct messages between their respective components and the other monitors in the network, and if needed they send dependency messages.
When a monitor detects an incorrect message, it signals an error to its component and the other monitors.
This way, eventually the entire network becomes aware of the protocol violation, and the execution stops.
It is then up to the components to gracefully deal with the protocol violation, e.g., by reverting to a prior state or restarting the protocol from the start.

The toolkit comes with two test suites:
\begin{itemize}
    \item
        The authorization protocol in $G_{\sff{a}}$~\eqref{eq:mpstMon:auth}, our running example.

    \item
        A weather protocol $G_{\sff{w}}$ between a client ($c$), a city database ($d$), and a weather API ($w$).
        For exchanges with a single branch, we write the message in parentheses and omit curly braces.
        \[
            c {!} w (\msg \sff{key}<\sff{str}>) \gtc. \mu X \gtc. c {!} d (\msg \sff{city}<\sff{str}>) \gtc. d {!} c \gtBraces*{
                \begin{array}{@{}l@{}}
                    \msg \sff{coord}<\sff{str}> \gtc. c {!} w (\msg \sff{coord}<\sff{str}>) \gtc. w {!} c (\msg \sff{temp}<\sff{real}>) \gtc. X \gtc, \\
                    \msg \sff{unknown}<> \gtc. X
                \end{array}
            }
        \]
        This is an interesting test suite, because the weather API (which requires an API $\sff{key}$) is not set up to deal with dependencies.
        The suite compensates by including a program that acts as a ``translator'' for the weather API.
        The system is then still protected from protocol violations by the weather API.
\end{itemize}

\newpage
\section{Definitions and Proofs}
\label{a:defsProofs}

\subsection{Relative Types with Locations (\Cref{s:globalTypesAsMonitors,s:correctMonitoredBlackboxes})}
\label{a:relativeTypesWithLocs}

Here, we formally define relative types with locations, and define how they are used in related definitions.

We refer to an ordered sequence of labels $\vect \ell$ as a \emph{location}.
We write $\mbb L$ to denote a set of locations.

\begin{definition}[Relative Types with Locations]
    \label{d:mpstMon:relativeTypesWithLocs}
    \emph{Relative types with locations} are defined by the following syntax:
    \[
        \rtc{R},\rtc{R'} ::=
        p {!} q^{\mbb L} \rtBraces{ \msg i<T_i> \rtc. \rtc{R} }_{i \in I}
        \sepr
        (p {!} r) {!} q^{\mbb L} \rtBraces{ i \rtc. \rtc{R} }_{i \in I}
        \sepr
        (p {?} r) {!} q^{\mbb L} \rtBraces{ i \rtc. \rtc{R} }_{i \in I}
        \sepr
        \mu X \rtc. \rtc{R}
        \sepr
        X
        \sepr
        \tEnd
    \]
\end{definition}

\begin{definition}[Relative Projection with Locations]
    \label{d:mpstMon:relativeProjectionWithLocs}
    \emph{Relative projection with locations}, denoted $\gtc{G} \wrt (p,q)^{\mbb L}$, is defined by \Cref{alg:mpstMon:relativeProjectionWithLocs}.
    This algorithm relies on three auxiliary definitions:
    \begin{itemize}
    	\item
            The \emph{erasure} of a relative type, denoted $\erase(\rtc{R})$, is defined by replacing each set of locations in $\rtc{R}$ by $\emptyset$ (e.g., $\erase(p {!} q^{\mbb L} \rtBraces{ \msg i<T_i> \rtc. \rtc{R_i} }_{i \in I}) \deq p {!} q^\emptyset \rtBraces{ \msg i<T_i> \rtc. \erase(\rtc{R_i}) }_{i \in I}$).

        \item
            Given $\rtc{R}$ and $\rtc{R'}$ such that $\erase(\rtc{R}) = \erase(\rtc{R'})$, we define the \emph{union} of $\rtc{R}$ and $\rtc{R'}$, denoted $\rtc{R} \cup \rtc{R'}$, by combining each set of locations for each corresponding message in $\rtc{R}$ and~$\rtc{R'}$ (e.g., $p {!} q^{\mbb L} \rtBraces{ \msg i<T_i> \rtc. \rtc{R_i} }_{i \in I} \cup p {!} q^{\mbb L'} \rtBraces{ \msg i<T_i> \rtc. \rtc{R'_i} }_{i \in I} \deq p {!} q^{\mbb L \cup \mbb L'} \rtBraces{ \msg i<T_i> \rtc. (\rtc{R_i} \cup \rtc{R'_i}) }_{i \in I}$).
            Given $(\rtc{R_a})_{a \in A}$ for finite $A$ such that, for each $a,b \in A$, $\erase(\rtc{R_a}) = \erase(\rtc{R_b})$, we inductively define $\bigcup_{a \in A} \rtc{R_a}$ as expected.

        \item
            We define the \emph{appendance} of a label to a set of locations, denoted $\mbb L + \ell$, as follows:
            \begin{align*}
                (\mbb L \cup \braces{ \vect \ell }) + \ell' &\deq (\mbb L + \ell') \cup \braces{ \vect \ell , \ell' }
                &
                \emptyset + \ell' &\deq \emptyset
            \end{align*}

            We extend this definition to the appendance of a location to a set of locations as follows:
            \begin{align*}
                \mbb L + (\ell' , \vect \ell) &\deq (\mbb L + \ell') + \vect \ell
                &
                \mbb L + \epsi &\deq \mbb L
            \end{align*}

            We extend this definition to the appendance of two sets of locations as follows:
            \begin{align*}
                \mbb L + \mbb L' &\deq \bigcup_{\vect \ell' \in \mbb L'} \mbb L + \vect \ell'
            \end{align*}
    \end{itemize}
\end{definition}

\begin{algorithm}[t]
    \DontPrintSemicolon
    \SetAlgoNoEnd
    \SetInd{.2em}{.7em}
    \SetSideCommentRight
    \Def{$\gtc{G} \wrt (p,q)^{\mbb L}$}{

        \Switch{$\gtc{G}$}{

            \uCase{$s {!} r \gtBraces{ \msg i<T_i> \gtc. \gtc{G_i} }_{i \in I}$}{
                $\forall i \in I.~ \rtc{R_i} \deq \gtc{G_i} \wrt (p,q)^{\mbb L + i}$ \;

                \lIf{$(p = s \wedge q = r)$}{
                    \KwRet
                    $p {!} q^{\mbb L} \rtBraces{ \msg i<T_i> \rtc. \rtc{R_i} }_{i \in I}$
                }

                \lElseIf{$(q = s \wedge p = r)$}{
                    \KwRet
                    $q {!} p^{\mbb L} \rtBraces{ \msg i<T_i> \rtc. \rtc{R_i} }_{i \in I}$
                }

                \lElseIf{$\forall i,j \in I.~ \erase(\rtc{R_i}) = \erase(\rtc{R_j})$}{
                    \KwRet
                    $\bigcup_{i \in I} \rtc{R_i}$
                }\label{li:mpstMon:rpIndepLab}

                \lElseIf{$s \in \braces{p,q} \wedge t \in \braces{p,q} \setminus \braces{s}$}{
                    \KwRet
                    $(s {!} r) {!} t^{\mbb L} \rtBraces{ i \rtc. \rtc{R_i} }_{i \in I}$
                }

                \lElseIf{$r \in \braces{p,q} \wedge t \in \braces{p,q} \setminus \braces{r}$}{
                    \KwRet
                    $(r {?} s) {!} t^{\mbb L} \rtBraces{ i \rtc. \rtc{R_i} }_{i \in I}$
                }
            }

            \uCase{$\mu X \gtc. \gtc{G'}$}{
                $\rtc{R'} \deq \gtc{G'} \wrt (p,q)^{\mbb L}$ \;

                \lIf{$(\text{$\rtc{R'}$ contains an exchange or a recursive call on any $Y \neq X$})$}{
                    \KwRet
                    $\mu X^{\mbb L} \rtc. \rtc{R'}$
                }

                \lElse{
                    \KwRet
                    $\tEnd$
                }
            }

            \lCase{$X$}{
                \KwRet
                $X^{\mbb L}$
            }

            \lCase{$\tEnd$}{
                \KwRet
                $\tEnd$
            }
        }

    }

    \caption{Relative projection with locations.}
    \label{alg:mpstMon:relativeProjectionWithLocs}
\end{algorithm}

\begin{definition}[Dependence with Locations]
    \label{d:mpstMon:depsOnWithLocs}
    Given a well-formed global type $\gtc{G}$, we say $p$'s role in $\gtc{G}$ \emph{depends on} $q$'s role in the initial exchange in $\gtc{G}$, denoted $\depsOn p q \gtc{G}$, if an only if
    \[
        \gtc{G} = s {!} r \gtBraces{ \msg i<T_i> \gtc. \gtc{G_i} }_{i \in I}
        \wedge p \notin \braces{s,r}
        \wedge q \in \braces{s,r}
        \wedge \exists i,j \in I.~ \erase(\gtc{G_i} \wrt (p,q)) \neq \erase(\gtc{G_j} \wrt (p,q)).
    \]
\end{definition}

\subsection{Unfolding Recursive Types and Monitors}

When unfolding recursive relative types, the location annotations require care.
Consider, for example
\[
    \mu X^{\braces{1,2}} \rtc. p {!} q^{\braces{1,2}} \rtBraces{ \msg \ell_1<T_1> \rtc. q {!} p^{\braces{1,2,\ell_1}} \rtBraces{ \msg \ell_2<T_2> \rtc. X^{\braces{1,2,\ell_1,\ell_2}} } \rtc, \msg \ell'<T'> \rtc. X^{\braces{1,2,\ell'}} }.
\]
To unfold this type, we should replace each recursive call on $X$ with a copy of the whole recursive definition.
However, this is insufficient: the locations of the original recursive definition (starting at $1,2$) do not concur with the locations of the recursive calls ($1,2,\ell_1,\ell_2$ and $1,2,\ell'$).
Thus, we need to update the locations of the copied recursive calls by inserting the new path behind the location of the original recursive definition.
This way, the recursive call at $1,2,\ell_1,\ell_2$ would get replaced by (inserted locations are underlined)
\[
    \mu X^{\braces{1,2,\underline{\ell_1,\ell_2}}} \rtc. p {!} q^{\braces{1,2,\underline{\ell_1,\ell_2}}} \rtBraces*{
        \begin{array}{@{}l@{}}
            \msg \ell_1<T_1> \rtc. q {!} p^{\braces{1,2,\underline{\ell_1,\ell_2},\ell_1}} \rtBraces{ \msg \ell_2<T_2> \rtc. X^{\braces{1,2,\underline{\ell_1,\ell_2},\ell_1,\ell_2}} } \rtc,
            \\
            \msg \ell'<T'> \rtc. X^{\braces{1,2,\underline{\ell_1,\ell_2},\ell'}}
        \end{array}
    }.
\]

We formally define unfolding of relative types, global types, and monitors, and prove some essential properties of these unfoldings.
First, we give an overview of these definitions and results:

\begin{itemize}
    \item
        We formally define the unfolding of a relative type $\mu X^{\mbb L} \rtc. \rtc{R}$ by removing the prefix $\mbb L$ from the locations in $\rtc{R}$ (using $\remPref$), replacing each recursive call $X^{\mbb L'}$ (where $\mbb L'$ is the location of the replaced recursive call) with the recursive definition beginning at location $\mbb L'$ (using $\prepend$), and then replacing the original location $\mbb L$ (using $\prepend$).
        \Cref{d:mpstMon:manipCompLocs} gives the required ingredients.

    \item
        \Cref{d:mpstMon:unfoldRT} then defines the unfolding of relative types.

    \item
        \Cref{l:mpstMon:prefProjection} states relations between relative projection, $\prepend$ and $\remPref$.

    \item
        \Cref{d:mpstMon:unfoldGT} defines the unfolding of global types.

    \item
        \Cref{l:mpstMon:unfoldRecursiveGTRTnoLabel,l:mpstMon:unfoldRecursiveGTRT,l:mpstMon:unfoldRT} equate unfoldings of relative projections and global types.

    \item
        \Cref{d:mpstMon:unfoldMon} defines the unfolding of monitors.

    \item
        \Cref{l:mpstMon:monitorUnfoldOne,l:mpstMon:activeUnfold,l:mpstMon:monitorDropInactive,l:mpstMon:monitorUnfold,l:mpstMon:monitorUnfoldEnd} equate unfoldings of synthesized monitors and global types.
\end{itemize}

\begin{definition}[Manipulation and Comparison of Locations]
    \label{d:mpstMon:manipCompLocs}
    \leavevmode
    \begin{itemize}

        \item
            Given locations $\vect \ell$ and $\vect \ell'$, we say $\vect \ell'$ is a \emph{prefix} of $\vect \ell$, denoted $\vect \ell' \prefixeq \vect \ell$, if there exists a \emph{suffix}~$\vect \ell''$ such that $\vect \ell = \vect \ell' , \vect \ell''$.

        \item
            We say $\vect \ell'$ is a \emph{strict prefix} of $\vect \ell$, denoted $\vect \ell' \prefix \vect \ell$, if $\vect \ell' \prefixeq \vect \ell$ with suffix $\vect \ell'' \neq \epsi$.
            We extend the prefix relation to sets of locations as follows: $\mbb L' \prefixeq \mbb L$ iff $\forall \vect \ell \in \mbb L.~ \exists \vect \ell' \in \mbb L'.~ \vect \ell' \prefixeq \vect \ell$, i.e., each location in $\mbb L$ is prefixed by a location in $\mbb L'$.

        \item
            Given $\vect \ell \neq \epsi$, we write $\fst(\vect \ell)$ to denote the first element of $\vect \ell$; formally, there exists $\vect \ell'$ such that $\fst(\vect \ell) , \vect \ell' = \vect \ell$.

        \item
            Given a set of locations $\mbb L$ and a relative type $\rtc{R}$, we define the \emph{prependance} of $\mbb L$ to the locations in $\rtc{R}$, denoted $\prepend( \mbb L , \rtc{R} )$, by prepending $\mbb L$ to each location in $\rtc{R}$ inductively; e.g.,
            \[
                \prepend( \mbb L , p {!} q^{\mbb L'} \rtBraces{ \msg i<T_i> \rtc. \rtc{R_i} }_{i \in I} ) \deq p {!} q^{\mbb L + \mbb L'} \rtBraces{ \msg i<T_i> \rtc. \prepend( \mbb L , \rtc{R_i} ) }_{i \in I}.
            \]

        \item
            Given a relative type $\rtc{R}$, we define its \emph{first location}, denoted $\fstLoc(\rtc{R})$, as the location annotation on the first exchange in $\rtc{R}$; e.g., $\fstLoc(p {!} q^{\mbb L} \rtBraces{ \msg i<T_i> \rtc. \rtc{R_i} }_{i \in I}) \deq \mbb L$.

        \item
            Given a set of locations $\mbb L$ and a relative type $\rtc{R}$, we define the \emph{removal of prefix} $\mbb L$ from the locations in $\rtc{R}$, denoted $\remPref( \mbb L , \rtc{R} )$.
            Formally, it checks each location in $\mbb L$ as a possible prefix of each location in the set of locations of each exchange of $\rtc{R}$, and leaves only the suffix; e.g., $\remPref( \mbb L , p {!} q^{\mbb L'} \rtBraces{ \msg i<T_i> \rtc. \rtc{R_i} }_{i \in I} ) \deq p {!} q^{\mbb L''} \rtBraces{ \msg i<T_i> \rtc. \remPref( \mbb L , \rtc{R_i} ) }_{i \in I}$, where $\mbb L'' \deq \braces{ \vect \ell'' \mid \exists \vect \ell \in \mbb L , \vect \ell' \in \mbb L'.~ \vect \ell \prefixeq \vect \ell' \text{ with suffix } \vect \ell'' }$.

    \end{itemize}
\end{definition}

\begin{definition}[Unfold Relative Type]
    \label{d:mpstMon:unfoldRT}
    Given a relative type $\mu X^{\mbb L} \rtc. \rtc{R}$, we define its \emph{one-level unfolding}, denoted $\unfold_1( \mu X^{\mbb L} \rtc. \rtc{R} )$, as follows:
    \[
        \unfold_1( \mu X^{\mbb L} \rtc. \rtc{R} ) \deq \prepend( \mbb L , \remPref( \mbb L , \rtc{R} ) \braces{ \prepend( \mbb L' , \remPref( \mbb L , \mu X^{\mbb L} \rtc. \rtc{R} ) ) / X^{\mbb L'} } )
    \]

    Given a relative type $\rtc{R}$, we define its \emph{(full) unfolding}, denoted $\unfold( \rtc{R} )$, as follows:
    \[
        \unfold( \rtc{R} )
        \deq \begin{cases}
            \unfold_1( \mu X^{\mbb L} \rtc. \unfold( \rtc{R'} ) )
            & \text{if $\rtc{R} = \mu X^{\mbb L} \rtc. \rtc{R'}$}
            \\
            \rtc{R}
            & \text{otherwise}
        \end{cases}
    \]
\end{definition}

\begin{lemma}
    \label{l:mpstMon:prefProjection}
    For any well-formed global type $\gtc{G}$, participants $p,q \in \prt(\gtc{G})$, and set of locations $\mbb L$,
    \begin{align*}
        \prepend( \mbb L , \gtc{G} \wrt (p,q)^{\braces{\epsi}} ) &= \gtc{G} \wrt (p,q)^{\mbb L},
        \\
        \remPref( \mbb L , \gtc{G} \wrt (p,q)^{\mbb L} ) &= \gtc{G} \wrt (p,q)^{\braces{\epsi}}.
    \end{align*}
\end{lemma}

\begin{proof}
    By definition.
\end{proof}

\begin{definition}[Unfold Global Type]
    \label{d:mpstMon:unfoldGT}
    Given a well-formed global type $\mu X \gtc. \gtc{G}$, we define its \emph{one-level unfolding}: $\unfold_1( \mu X \gtc. \gtc{G} ) \deq \gtc{G} \braces{ \mu X \rtc. \gtc{G} / X }$.
    Given a well-formed global type $\gtc{G}$, we define its \emph{(full) unfolding}, denoted $\unfold(\gtc{G})$, as follows:
    \[
        \unfold(\gtc{G}) \deq \begin{cases}
            \unfold_1( \mu X \gtc. \unfold( \gtc{G'} ) )
            & \text{if $\gtc{G} = \mu X \gtc. \gtc{G'}$}
            \\
            \gtc{G}
            & \text{otherwise}
        \end{cases}
    \]
\end{definition}

\begin{lemma}
    \label{l:mpstMon:unfoldRecursiveGTRTnoLabel}
    For any well-formed global type $\mu X \gtc. \gtc{G}$ and participants $p,q \in \prt(\gtc{G})$,
    \[
        (\gtc{G} \wrt (p,q)^{\braces{\epsi}}) \braces{ ((\mu X \gtc. \gtc{G}) \wrt (p,q)^{\mbb L'} ) / X^{\mbb L'} }
        =
        (\gtc{G} \braces{ \mu X \gtc. \gtc{G} / X }) \wrt (p,q)^{\braces{\epsi}}.
    \]
\end{lemma}

\begin{proof}
    By definition.
    The starting location $\mbb L'$ for replacing each $X^{\mbb L}$ is correct on the right-hand-side, because $\mbb L'$ is the location of the recursive call.
    Hence, the projection of the unfolded global type will at those spots start with $\mbb L'$.
\end{proof}

\begin{lemma}
    \label{l:mpstMon:unfoldRecursiveGTRT}
    For any well-formed global type $\mu X \gtc. \gtc{G}$, participants $p,q \in \prt(\gtc{G})$, and set of locations $\mbb L$,
    \[
        \unfold_1( \mu X \gtc. \gtc{G} \wrt (p,q)^{\mbb L} )
        =
        \unfold_1( \mu X \gtc. \gtc{G} ) \wrt (p,q)^{\mbb L}.
    \]
\end{lemma}

\begin{proof}
    By \Cref{d:mpstMon:unfoldRT,l:mpstMon:prefProjection,l:mpstMon:unfoldRecursiveGTRTnoLabel}:
    \begin{align*}
        & \unfold_1( \mu X \gtc. \gtc{G} \wrt (p,q)^{\mbb L} )
        \\
        \overset{\text{\Cref{d:mpstMon:relativeProjectionWithLocs}}}{=} & \unfold_1( \mu X^{\mbb L} \rtc. (\gtc{G} \wrt (p,q)^{\mbb L}) )
        \\
        \overset{\text{\Cref{d:mpstMon:unfoldRT}}}{=} & \prepend( \mbb L , \remPref( \mbb L , \gtc{G} \wrt (p,q)^{\mbb L} ) \\
        & \quad \braces{ \prepend( \mbb L' , \remPref( \mbb L , \mu X^{\mbb L} \rtc. (\gtc{G} \wrt (p,q)^{\mbb L}) ) ) / X^{\mbb L'} } )
        \\
        \overset{\text{\Cref{d:mpstMon:relativeProjectionWithLocs}}}{=} & \prepend( \mbb L , \remPref( \mbb L , \gtc{G} \wrt (p,q)^{\mbb L} ) \braces{ \prepend( \mbb L' , \remPref( \mbb L , \mu X \gtc. \gtc{G} \wrt (p,q)^{\mbb L} ) ) / X^{\mbb L'} } )
        \\
        \overset{\text{\Cref{l:mpstMon:prefProjection}}}{=} & \prepend( \mbb L , (\gtc{G} \wrt (p,q)^{\braces{\epsi}}) \braces{ \prepend( \mbb L' , \mu X \gtc. \gtc{G} \wrt (p,q)^{\braces{\epsi}} ) / X^{\mbb L'} } )
        \\
        \overset{\text{\Cref{l:mpstMon:prefProjection}}}{=} & \prepend( \mbb L , (\gtc{G} \wrt (p,q)^{\braces{\epsi}}) \braces{ (\mu X \gtc. \gtc{G} \wrt (p,q)^{\mbb L'}) / X^{\mbb L'} } )
        \\
        \overset{\text{\Cref{l:mpstMon:unfoldRecursiveGTRTnoLabel}}}{=} & \prepend( \mbb L , \gtc{G} \braces{ \mu X \gtc. \gtc{G} / X } \wrt (p,q)^{\braces{\epsi}} )
        \\
        \overset{\text{\Cref{l:mpstMon:prefProjection}}}{=} & \gtc{G} \braces{ \mu X \gtc. \gtc{G} / X } \wrt (p,q)^{\mbb L}
        \\
        \overset{\text{\Cref{d:mpstMon:unfoldGT}}}{=} & \unfold_1( \mu X \gtc. \gtc{G} ) \wrt (p,q)^{\mbb L}
        \tag*{\qedhere}
    \end{align*}
\end{proof}

\begin{lemma}
    \label{l:mpstMon:unfoldRT}
    For any well-formed global type $\gtc{G}$, participants $p,q \in \prt(\gtc{G})$, and set of locations $\mbb L$,
    \[
        \unfold( \gtc{G} \wrt (p,q)^{\mbb L} ) = \unfold( \gtc{G} ) \wrt (p,q)^{\mbb L}.
    \]
\end{lemma}

\begin{proof}
    By induction on the number of recursive definitions that $\gtc{G}$ starts with (finite by well-formedness, \Cref{d:mpstMon:wf}).
    In the base case, the thesis follows trivially.
    In the inductive case, $\gtc{G} = \mu X \gtc. \gtc{G'}$:
    \begin{align*}
        & \unfold( \mu X \gtc. \gtc{G'} \wrt (p,q)^{\mbb L} )
        \\
        \overset{\text{\Cref{d:mpstMon:relativeProjectionWithLocs}}}{=} & \unfold( \mu X^{\mbb L} \rtc. (\gtc{G'} \wrt (p,q)^{\mbb L}) )
        \\
        \overset{\text{\Cref{d:mpstMon:unfoldRT}}}{=} & \unfold_1( \mu X^{\mbb L} \rtc. \unfold( \gtc{G'} \wrt (p,q)^{\mbb L} ) )
        \\
        \overset{\text{IH}}{=} & \unfold_1( \mu X^{\mbb L} \rtc. (\unfold( \gtc{G'} ) \wrt (p,q)^{\mbb L}) )
        \\
        \overset{\text{\Cref{d:mpstMon:relativeProjectionWithLocs}}}{=} & \unfold_1( \mu X \gtc. \unfold( \gtc{G'} ) \wrt (p,q)^{\mbb L} )
        \\
        \overset{\text{\Cref{l:mpstMon:unfoldRecursiveGTRT}}}{=} & \unfold_1( \mu X \gtc. \unfold( \gtc{G'} ) ) \wrt (p,q)^{\mbb L}
        \\
        \overset{\text{\Cref{d:mpstMon:unfoldGT}}}{=} & \unfold( \mu X \gtc. \gtc{G'} ) \wrt (p,q)^{\mbb L}
        \tag*{\qedhere}
    \end{align*}
\end{proof}

\begin{definition}[Unfold Monitor]
    \label{d:mpstMon:unfoldMon}
    Given a monitor $\mu X \mc. \mc{M}$, we define its \emph{one-level unfolding}: $\unfold_1( \mu X \mc. \mc{M} ) \deq \mc{M} \braces{ \mu X \mc. \mc{M} / X }$.
    Given a monitor $\mc{M}$, we define its \emph{(full) unfolding}, denoted $\unfold( \mc{M} )$, as follows:
    \[
        \unfold( \mc{M} ) \deq \begin{cases}
            \unfold_1( \mu X \mc. \unfold( \mc{M} ) )
            & \text{if $\mc{M} = \mu X \mc. \mc{M'}$}
            \\
            \mc{M}
            & \text{otherwise}
        \end{cases}
    \]
\end{definition}

\begin{lemma}
    \label{l:mpstMon:monitorUnfoldOne}
    Suppose given a well-formed global type $\gtc{G} = \mu X \gtc. \gtc{G'}$, a participant $p$, and a set of locations $\mbb L$.
    Let $D \deq \braces{ q \in \prt(\gtc{G}) \setminus \braces{p} \mid \gtc{G} \wrt (p,q)^{\mbb L} \neq \tEnd }$.
    Then
    \[
        \unfold_1( \mu X \mc. \gtToMon( \gtc{G'} , p , D ) ) = \gtToMon( \unfold_1( \mu X \gtc. \gtc{G'} ) , p , D ).
    \]
\end{lemma}

\begin{proof}
    By definition.
    The recursive calls in $\mc{M'}$ concur with the recursive calls in $\gtc{G'}$, and are prefixed by all the exchanges in $\gtc{G'}$ in which $p$ is involved.
\end{proof}

\begin{lemma}
    \label{l:mpstMon:activeUnfold}
    Suppose given
    \begin{itemize}
        \item a well-formed global type $\mu X \gtc. \gtc{G}$,
        \item a set of participants $D$,
        \item a participant $p \notin D$, and
        \item a set of locations $\mbb L$.
    \end{itemize}
    Let $E_1 \deq \braces{ q \in D \mid (\mu X \gtc. \gtc{G}) \wrt (p,q)^{\mbb L} \neq \tEnd }$ and $E_2 \deq \braces{ q \in D \mid (\gtc{G} \braces{ \mu X \gtc. \gtc{G} / X }) \wrt (p,q)^{\mbb L} \neq \tEnd }$.

    Then $E_1 = E_2$.
\end{lemma}

\begin{proof}
    For any $q \in E_1$, $(\mu X \gtc. \gtc{G}) \wrt (p,q)^{\mbb L} \neq \tEnd$.
    Then, by definition, $G \wrt (p,q)^{\mbb L} \neq \tEnd$.
    Hence, the projection of the unfolding of $\gtc{G}$ is also not $\tEnd$, and thus $q \in E_2$.

    For any $q \in E_2$, the projection of the unfolding of $\gtc{G}$ is not $\tEnd$.
    Then, by definition, $\gtc{G} \wrt (p,q)^{\mbb L} \neq \tEnd$.
    Hence, by definition, $(\mu X \gtc. \gtc{G}) \wrt (p,q)^{\mbb L} \neq \tEnd$, and thus $q \in E_1$.
\end{proof}

\begin{lemma}
    \label{l:mpstMon:monitorDropInactive}
    Suppose given
    \begin{itemize}
        \item a well-formed global type $\gtc{G}$,
        \item a set of participants $D$,
        \item a participant $p \notin D$, and
        \item a set of locations $\mbb L$.
    \end{itemize}
    Let $D' \deq \braces{ q \in D \mid \gtc{G} \wrt (p,q)^{\mbb L} \neq \tEnd }$.

    Then $\gtToMon( \gtc{G} , p , D ) = \gtToMon( \gtc{G} , p , D' )$.
\end{lemma}

\begin{proof}
    By definition.
    Since $p$ does not interact with any $q \in D \setminus D'$, only the $q \in D'$ affect the creation of the monitor.
\end{proof}

\begin{lemma}
    \label{l:mpstMon:monitorUnfold}
    Suppose given
    \begin{itemize}
        \item a well-formed global type $\gtc{G}$,
        \item a participant $p$, and
        \item a set of locations $\mbb L$.
    \end{itemize}
    Let $D \deq \prt(\gtc{G}) \setminus \braces{p}$.

    Then $\unfold( \gtToMon( \gtc{G} , p , D ) ) = \gtToMon( \unfold( \gtc{G} ) , p , D )$.
\end{lemma}

\begin{proof}
    By induction on the number of recursive definitions that $\gtc{G}$ starts with (finite by well-formedness, \Cref{d:mpstMon:wf}).
    In the base case, the thesis follows trivially.
    In the inductive case, $\gtc{G} = \mu X \gtc. \gtc{G'}$.
    Let $D' \deq \braces{ q \in D \mid \gtc{G'} \wrt (p,q)^\epsi \neq \tEnd }$.
    \begin{align*}
        & \unfold( \gtToMon( \mu X \gtc. \gtc{G'} , p , D ) )
        \\
        \overset{\text{\Cref{d:mpstMon:gtToMon}}}{=} & \unfold( \mu X \mc. \gtToMon( \gtc{G'} , p , D' ) )
        \\
        \overset{\text{\Cref{d:mpstMon:unfoldMon}}}{=} & \unfold_1( \mu X \mc. \unfold( \gtToMon( \gtc{G'} , p , D' ) ) )
        \\
        \overset{\text{IH}}{=} & \unfold_1( \mu X \mc. \gtToMon( \unfold( \gtc{G'} ) , p , D' ) )
        \\
        \overset{\text{\Cref{l:mpstMon:monitorUnfoldOne}}}{=} & \gtToMon( \unfold_1 ( \mu X \gtc. \unfold( \gtc{G'} ) ) , p , D' )
        \\
        \overset{\text{\Cref{d:mpstMon:unfoldGT}}}{=} & \gtToMon( \unfold( \mu X \gtc. \gtc{G'} ) , p , D' )
        \\
        \overset{\text{\Cref{l:mpstMon:activeUnfold,l:mpstMon:monitorDropInactive}}}{=} & \gtToMon( \unfold( \mu X \gtc. \gtc{G'} ) , p , D )
        \tag*{\qedhere}
    \end{align*}
\end{proof}

\begin{lemma}
    \label{l:mpstMon:monitorUnfoldEnd}
    Suppose given
    \begin{itemize}
        \item a well-formed global type $\gtc{G} = \mu X \gtc. \gtc{G'}$,
        \item a set of participants $D$, and
        \item a participant $p \notin D$.
    \end{itemize}
    If $\gtToMon( \gtc{G} , p , D ) = \tEnd$, then $\gtToMon( \gtc{G'}, p , D ) \braces{ \gtToMon( \gtc{G'} , p , D ) / X } = \tEnd$.
\end{lemma}

\begin{proof}
    Since $\gtToMon( \gtc{G} , p , D ) = \tEnd$, we have that $\gtc{G} \wrt (p,q)^{\braces{\epsi}} = \tEnd$ for every $q \in D$.
    It follows by definition that $\gtToMon( \gtc{G'} , p , D ) = \tEnd$.
    Hence, the unfolding is also $\tEnd$.
\end{proof}

\subsection{Proof of Soundness}
\label{s:soundnessProof}

Here we prove \Cref{t:mpstMon:soundness}, which is a generalized version of \Cref{t:mpstMon:soundnessMain} (\Cpageref{t:mpstMon:soundnessMain}).
We start with an overview of intermediate results used for the proof:
\begin{itemize}

    \item
        \Cref{l:mpstMon:independence} shows that transitions that do not affect parallel networks are independent, i.e., they can be executed in any order without changing the outcome.

    \item
        \Cref{l:mpstMon:satisfactionDropLabel} shows that we can empty the $\Lbls$ map of a satisfaction relation if the locations of all the relative types in $\RTs$ succeed all labels in $\dom(\Lbls)$, i.e., no exchanges in relative types in $\RTs$ relate to the choices recorded in $\Lbls$.

    \item
        \Cref{l:mpstMon:satisfactionDropUnion} shows that we can eliminate unions of relative types in the $\RTs$ map of a satisfaction relation for participants that do not depend on some exchange, allowing us to specifize those independent relative types to some chosen branch.

    \item
        \Cref{d:mpstMon:initSatisfied} defines a relation between global type, participants, relative types, monitor, blackbox, and buffer, such that together they are witness to a satisfaction relation.

    \item
        \Cref{l:mpstMon:soundnessGeneralized} shows an inductive variant of soundness, given an intermediate global type $\gtc{G_0}$ between the initial $\gtc{G}$ and the final $\gtc{G'}$.

    \item
        \Cref{t:mpstMon:soundness} shows soundness.

    \item
        \Cref{t:mpstMon:errorFreedom} shows error freedom.

\end{itemize}

\begin{lemma}[Independence]
    \label{l:mpstMon:independence}
    Suppose $\net P_1 \| \net Q \| \net R \ltrans{ \tau } \net P'_1 \| \net Q' \| \net R$ and $\net P_2 \| \net Q \| \net R \ltrans{ \tau } \net P'_2 \| \net Q \| \net R'$, where $\net P_1 \not\equiv \net P'_1$, $\net P_2 \not\equiv \net P'_2$, $\net Q \not\equiv \net Q'$, and $\net R \not\equiv \net R'$.
    Then
    \[
        \begin{array}{@{}c@{}c@{}c@{}}
            \net P_1 \| \net P_2 \| \net Q \| \net R
            &
            {} \ltrans{ \tau } {}
            &
            \net P'_1 \| \net P_2 \| \net Q' \| \net R
            \\
            \downarrow \mkern-4mu {\scriptstyle \tau}
            &
            &
            \downarrow \mkern-4mu {\scriptstyle \tau}
            \\
            \net P_1 \| \net P'_2 \| \net Q \| \net R'
            &
            {} \ltrans{ \tau } {}
            &
            \net P'_1 \| \net P'_2 \| \net Q' \| \net R'.
        \end{array}
    \]
\end{lemma}

\begin{proof}
    By definition of the LTS for networks (\Cref{d:mpstMon:ltsNetworks}) the $\tau$-transitions of both premises are derived from applications of Transition~\ruleLabel{par} and an application of Transition~\ruleLabel{out-buf} or Transition~\ruleLabel{out-mon-buf}.
    For example, in the first premise, $\net P_1$ does an output which ends up in the buffer of a (monitored) blackbox in $\net Q$, leaving $\net R$ unchanged.
    In the second premise, $\net P_2$ does an output which ends up in a buffer in $\net R$, leaving $\net Q$ unchanged.
    Hence, the outputs by $\net P_1$ and $\net P_2$ have completely different senders and recipients.
    As a result, in a network with all of $\net P_1$, $\net P_2$, $\net Q$, and $\net R$ these exchanges do not influence each other.
    The conclusion is that the order of these exchanges does not matter.
\end{proof}

\begin{lemma}
    \label{l:mpstMon:satisfactionDropLabel}
    Suppose
    \begin{itemize}
        \item $\satisfies{ \net P }[ \Lbls ]{ \RTs }[ p ]$ and
        \item that, for every $( q , \rtc{R} ) \in \RTs$, $\rtc{R} \neq \tEnd$ implies $\big( \bigcup_{( \mbb L , \ell ) \in \Lbls} \mbb L \big) \prefix \fstLoc(\rtc{R})$.
    \end{itemize}
    Then $\satisfies{ \net P }{ \RTs }[ p ]$.
\end{lemma}

\begin{proof}
    By definition: clearly, if all the relative types are at a location past the recorded choices in $\Lbls$, none of the choices in $\Lbls$ will ever be used for satisfaction anymore.
\end{proof}

As satisfaction iterates through a collection of relative types obtained from a global type, there will be instances where some relative types are independent of the global type's initial exchange.
By relative projection with locations (\Cref{d:mpstMon:relativeProjectionWithLocs}), these relative types are the union of the relative projections of the branches of the exchange.
The following lemma assures that we can drop this union and simply continue with the relative projection of the branch followed by the participants involved in the exchange (sender, recipient, and/or depending participants).

\begin{lemma}
    \label{l:mpstMon:satisfactionDropUnion}
    Suppose given
    \begin{itemize}
        \item a well-formed global type $\gtc{G} = s {!} r \gtBraces{ \msg i<T_i> \gtc. \gtc{G_i} }_{i \in I}$, and
        \item $D \supseteq \prt(\gtc{G}) \setminus \braces{p}$.
    \end{itemize}
    Let $D' \deq \braces{q \in D \mid \braces{p,q} \subseteq \braces{s,r}} \cup \braces{q \in D \mid \braces{p,q} \overlap \braces{s,r} \wedge (\depsOn q p \gtc{G} \vee \depsOn p q \gtc{G})}$.

    Suppose $\satisfies{ \net P }{ \RTs }[ p ]$, where
    \[
        \RTs = \braces{ ( q , \gtc{G_j} \wrt (p,q)^{\braces{\vect \ell,j}} ) \mid q \in D' } \cup \braces{ (q , \bigcup_{i \in I} (\gtc{G_i} \wrt (p,q)^{\braces{\vect \ell,i}}) ) \mid q \in D \setminus D' }
    \]
    for some $j \in I$.
    Let $\RTs' \deq \RTs \update{q \mapsto \gtc{G_j} \wrt (p,q)^{\braces{\vect \ell,j}}}_{q \in D \setminus D'}$.

    Then $\satisfies{ \net P }{ \RTs' }[ p ]$.
\end{lemma}

\begin{proof}
    For any $q \in D \setminus D'$, for every $i,k \in I$, $\erase(\gtc{G_i} \wrt (p,q)^{\braces{\vect \ell,k}}) = \erase(\gtc{G_k} \wrt (p,q)^{\braces{\vect \ell,k}})$.
    Since the satisfaction holds for an empty label function (i.e., $\satisfies{ \net P }[ \emptyset ]{ \RTs }[ p ]$), the locations formed in each $\gtc{G_i} \wrt (p,q)^{\braces{\vect \ell,i}}$ are insignificant.
    Hence, it suffices to simply use $\gtc{G_j} \wrt (p,q)^{\braces{\vect \ell,j}}$ for every $q \in D \setminus D'$.
\end{proof}

\begin{definition}[Initial satisfaction]
    \label{d:mpstMon:initSatisfied}
    We define \emph{initial satisfaction}, denoted \mbox{$( \gtc{G} , p , \mbb L , D ) \initSatisfied ( \RTs , D' , \mc{M} , P , \vect m )$}, to hold if and only if
    \begin{align*}
        D &\supseteq \prt(\gtc{G}) \setminus \braces{ p },
        \\
        \RTs &= \braces{ ( q , \gtc{G} \wrt (p,q)^{\mbb L} ) \mid q \in D },
        \\
        D' &= \braces{ q \in D \mid \unfold\big(\RTs(q)\big) \neq \tEnd },
        \\
        \gtToMon( \gtc{G} , p , D' ) \neq \tEnd
        &\implies \mc{M} = \gtToMon( \gtc{G} , p , D' )
        \\
        \gtToMon( \gtc{G}, p , D' ) = \tEnd
        &\implies \mc{M} \in \braces{ \tEnd , \checkmark }
        \\
        &\satisfies{ \monImpl{ \bufImpl{ p }{ P }{ \vect m } }{ \mc{M} }{ \epsi } }{ \RTs }[ p ]
    \end{align*}
\end{definition}

\begin{lemma}[Soundness --- Generalized]
    \label{l:mpstMon:soundnessGeneralized}
    \leavevmode
    \begin{itemize}
        \item Suppose given well-formed global types $\gtc{G}$ and $\gtc{G_0}$ such that $\prt(\gtc{G}) \geq 2$ and $\gtc{G} \ltrans{ \vect \ell_0 } \gtc{G_0}$.
        \item For every $p \in \prt(\gtc{G})$,
            \begin{itemize}
                \item suppose given a process $P_p^0$ and buffer $\vect m_p^0$, and
                \item take $\RTs_p^0 , D_p^0 , \mc{M_p^0}$ such that
                    \[
                        ( \gtc{G_0} , p , \braces{ \vect \ell_0 } , \prt(\gtc{G}) \setminus \braces{p} ) \initSatisfied ( \RTs_p^0 , D_p^0 , \mc{M_p^0} , P_p^0 , \vect m_p^0 ).
                    \]
            \end{itemize}
        \item Suppose $\net P_0 \deq \prod_{p \in \prt(\gtc{G})} \monImpl{ \bufImpl{ p }{ P_p^0 }{ \vect m_p^0 } }{ \mc{M_p^0} }{ \epsi } \trans* \net P'$.
    \end{itemize}
    Then there exist $\vect \ell',\gtc{G'}$ such that
    \begin{itemize}
        \item $\gtc{G_0} \ltrans{ \vect \ell' } \gtc{G'}$,
        \item for every $p \in \prt(\gtc{G})$ there exist $P'_p,\vect m'_p,\RTs'_p,D'_p,\mc{M'_p}$ such that
            \[
                ( \gtc{G'} , p , \braces{ \vect \ell_0 , \vect \ell' } , \prt(\gtc{G}) \setminus \braces{p} ) \initSatisfied ( \RTs'_p , D'_p , \mc{M'_p} , P'_p , \vect m'_p ),
            \]
            and
        \item $\net P' \trans* \prod_{p \in \prt(\gtc{G})} \monImpl{ \bufImpl{ p }{ P'_p }{ \vect m'_p } }{ \mc{M'_p} }{ \epsi }$.
    \end{itemize}
\end{lemma}

\begin{proof}
    By induction on the number of transitions $n$ from $\net P_0$ to $\net P'$ (\ih{1}).

    In the base case, where $n = 0$, the thesis follows immediately, because $\net P' = \net P_0$.
    To be precise, the assumption satisfies the conclusion by letting: $\gtc{G'} = \gtc{G_0}$; $\vect \ell = \epsi$; for every $p \in \prt(\gtc{G})$, $P'_p = P_p^0$, $\vect m'_p = \vect m_p^0$, $D'_p = D_p^0$, $\mc{M'_p} = \mc{M_p^0}$, $\RTs'_p = \RTs_p^0$.

    In the inductive case, where $n \geq 1$, we use the shape of $\gtc{G_0}$ and the assumed satisfactions of the monitored processes to determine the possible transitions from $\net P_0$.
    We then follow these transitions, and show that we can reach a network $\net P_1$ where all the monitored processes satisfy some $\gtc{G_1}$ with $\gtc{G_0} \ltrans{ \ell_1 } \gtc{G_1}$.
    If at this point we already passed through $\net P'$, the thesis is proven.
    Otherwise, the thesis follows from \ih{1}, because the number of transitions from $\net P_1$ is less than $n$.

    There is a subtlety that we should not overlook: $\gtc{G_0}$ may contain several consecutive, independent exchanges.
    For example, suppose $\gtc{G_0} = p {!} q \gtBraces{ \msg i<T_i> \gtc. s {!} r \gtBraces{ \msg j <T_j> \gtc. \gtc{G_{i,j}} }_{j \in J} }_{i \in I}$ where $\braces{p,q} \disjoint \braces{s,r}$.
    Monitors nor satisfaction can prevent the exchange between $s$ and $r$ from happening before the exchange between $p$ and $q$ has been completed.
    Hence, the transitions from $\net P_0$ to $\net P'$ may not entirely follow the order specified by $\gtc{G_0}$.

    We deal with this issue by applying induction on the number of out-of-order exchanges observed in the transitions to $\net P'$.
    We then follow the transitions determined by monitors and satisfaction, in the order specified by $\gtc{G_0}$, effectively ``postponing'' the out-of-order exchanges until it is their turn.
    We keep doing this, until we have eventually passed through all the postponed out-of-order exchanges.
    At this point, we have found an alternative path from $\net P_0$ to the final network.
    To reconcile this alternative path with the path from $\net P_0$ to $\net P'$, we apply independence (\Cref{l:mpstMon:independence}).
    This lemma essentially states that independent exchanges may be performed in any order, as they do not influence each other.
    Hence, we use independence (\Cref{l:mpstMon:independence}) to move the postponed transitions back to their original position in the path from $\net P_0$ to $\net P'$, proving the thesis.
    Hereafter, we assume independent exchanges dealt with.

    As a first step in our analysis, we consider the fact that $\gtc{G_0}$ may start with recursive definitions.
    Let $\gtc{G_1} \deq \unfold(\gtc{G_0})$; by definition, $\gtc{G_1}$ does not start with recursive definitions.
    For every $p \in \prt(\gtc{G})$, let $\mc{M_p^1} \deq \unfold(\mc{M_p^0})$; by \Cref{l:mpstMon:monitorUnfold}, $\mc{M_p^1} = \gtToMon(\gtc{G_1},p,D_p^0)$, and, by definition of the LTS for networks (\Cref{d:mpstMon:ltsNetworks}), the LTS of $\monImpl{ \bufImpl{ p }{ P_p^0 }{ \vect m_p^0 } }{ \mc{M_p^1} }{ \epsi }$ is equivalent to that of $\monImpl{ \bufImpl{ p }{ P_p^0 }{ \vect m_p^0 } }{ \mc{M_p^0} }{ \epsi }$.
    For every $p \in \prt(\gtc{G})$, let $\RTs_p^1 \deq \braces{ ( q , \unfold(\RTs_p^0(q)) ) \mid q \in D_p^0 }$; by \Cref{l:mpstMon:unfoldRT}, for every $q \in D_p^0$, $\RTs_p^1(q) = \gtc{G_1} \wrt (p,q)$.
    The conditions for satisfaction in \Cref{f:mpstMon:satisfaction} unfold any relative types.
    Hence, we can reuse the satisfaction given by the original initial satisfaction, to show that $(\gtc{G_1},p,\braces{\vect l_0},\prt(\gtc{G}) \setminus \braces{p}) \initSatisfied (\RTs_p^1,D_p^0,\mc{M_p^1},P_p^0,\vect m_p^0)$.
    We then continue our analysis from this new unfolded initial satisfaction, for which all results transfer back to the original initial satisfaction.

    The rest of our analysis depends on the shape of $\gtc{G_1}$ (exchange or end).
    \begin{itemize}
        \item
            Exchange: $\gtc{G_1} = s {!} r \gtBraces{ \msg i<T_i> \gtc. \gtc{H_i} }_{i \in I}$.

            Let us make an inventory of all relative types and monitors at this point.
            We use this information to determine the possible behavior of the monitored blackboxes, and their interactions.
            This behavior, in combination with satisfaction, allows us to determine exactly how the network evolves and reaches a state required to apply \ih1.

            \begin{itemize}
                \item
                    For $s$ we have:
                    \begin{align*}
                        \RTs_s^1(r)
                        &= s {!} r^{\braces{\vect \ell_0}} \rtBraces{ \msg i<T_i> \rtc. (\gtc{H_i} \wrt (s,r)^{\braces{\vect \ell_0,i}}) }_{i \in I}
                        \\
                        \depsVar_s
                        &= \braces{ q \in \prt(\gtc{G}) \mid \depsOn q s \gtc{G_1} }
                        \\
                        \RTs_s^1(q)
                        &= (s {!} r) {!} q^{\braces{\vect \ell_0}} \rtBraces{ i \rtc. (\gtc{H_i} \wrt (s,q)^{\braces{\vect \ell_0,i}}) }_{i \in I}
                        \quad [q \in \depsVar_s]
                        \\
                        \RTs_s^1(q)
                        &= \bigcup_{i \in I} (\gtc{H_i} \wrt (s,q)^{\braces{\vect \ell_0,i}})
                        \quad [q \in \prt(\gtc{G}) \setminus \braces{s,r} \setminus \depsVar_s]
                        \\
                        \mc{M_s^1}
                        &= s {!} r \mBraces{ \msg i<T_i> \mc. s {!} \depsVar_s (i) \mc. \gtToMon( \gtc{H_i} , s , D_s^0 ) }_{i \in I}
                    \end{align*}

                    \satref{sOutput}{Output} allows the monitored blackbox of $s$ to send $\msg j<T_j>$ for any $j \in I$ to $r$.
                    Then \satref{sDependencyOutput}{Dependency Output} allows the monitored blackbox to send $j$ to all $q \in \depsVar_s$ (concurrently).
                    However, there may be other relative types in $\RTs_s^1$ that allow/require the monitored blackbox of to perform other tasks.

                    Since the outputs above precede any other communications in $\RTs_s^1$ (they originate from the first exchange in $\gtc{G_1}$), by the progress property of satisfaction (\Cref{d:mpstMon:satisfaction}), the monitored blackbox will keep transitioning.
                    The monitor~$\mc{M_s^1}$ requires the blackbox to first perform the output above, and then the dependency outputs above.
                    The buffered blackbox does not take any transitions other than $\tau$ and the outputs above: otherwise, the monitor would transition to an error signal, which cannot transition, contradicting the satisfaction's progress property.

                    It then follows that the monitored blackbox sends $\msg j<T_j>$ to $r$ for some $j \in I$, after which it sends $j$ to each $q \in \depsVar_s$ (in any order), possibly interleaved with $\tau$-transitions from the blackbox (finitely many, by \bref{bTau}{Finite~$\tau$}) or from the buffered blackbox reading messages; let $P_s^2$ and $\vect m_s^2$ be the resulting blackbox and buffer, respectively.
                    By satisfaction (\Cref{d:mpstMon:satisfaction}), we have the following:
                    \begin{align*}
                        \RTs_s^2
                        &\deq \RTs_s^1
                        \update{q \mapsto \gtc{H_j} \wrt (s,q)^{\braces{\vect \ell_0,j}}}_{q \in \braces{r} \cup \depsVar_s}
                        \\
                        \mc{M_s^2}
                        &\deq \gtToMon( \gtc{H_j} , s , D_s^0 )
                        \\
                        &\satisfies{ \monImpl{ \bufImpl{ s }{ P_s^2 }{ \vect m_s^2 } }{ \mc{M_s^2} }{ \epsi } }[ \braces{ (\braces{\vect \ell_0} , j) } ]{ \RTs_s^2 }[ s ]
                    \end{align*}
                    Clearly, for each $q \in \dom(\RTs_s^2)$, $\vect \ell_0 \leq \fstLoc( \RTs_s^2(q) )$.
                    Hence, by \Cref{l:mpstMon:satisfactionDropLabel}, $\satisfies{ \monImpl{ \bufImpl{ s }{ P_s^2 }{ \vect m_s^2 } }{ \mc{M_s^2} }{ \epsi } }{ \RTs_s^2 }[ s ]$.
                    Let $\RTs_s^3 \deq \RTs_s^2 \update{q \mapsto \gtc{H_j} \wrt (r,q)^{\braces{\vect \ell_0,j}}}_{q \in D \setminus \depsVar_s \setminus \braces{s,r}}$.
                    Then, by \Cref{l:mpstMon:satisfactionDropUnion}, $\satisfies{ \monImpl{ \bufImpl{ s }{ P_s^2 }{ \vect m_s^2 } }{ \mc{M_s^2} }{ \epsi } }{ \RTs_s^3 }[ s ]$.
                    Let
                    \[
                        D_s^3 \deq \braces{ q \in \prt(\gtc{G}) \setminus \braces{s} \mid \unfold\big(\RTs_s^3(q)\big) \neq \tEnd }
                    \]
                    and $M_s^3 \deq \gtToMon( \gtc{H_j} , p , D_s^1 )$.
                    By \Cref{l:mpstMon:monitorDropInactive}, \mbox{$\mc{M_s^2} = \mc{M_s^3}$}.
                    We have $\gtc{G_1} \ltrans{j} \gtc{H_j}$.
                    In conclusion,
                    \[
                        ( \gtc{H_j} , s , \braces{\vect \ell_0,j} , \prt(\gtc{G}) \setminus \braces{s} ) \initSatisfied ( \RTs_s^3 , D_s^3 , \mc{M_s^3} , P_s^2 , \vect m_s^2 ),
                    \]
                    such that the premise of \ih1 is satisfied for $s$.

                \item
                    For $r$ we have:
                    \begin{align*}
                        \RTs_r^1(s)
                        &= s {!} r^{\braces{\vect \ell_0}} \rtBraces{ \msg i<T_i> \rtc. (\gtc{H_i} \wrt (r,s)^{\braces{\vect \ell_0,i}}) }_{i \in I}
                        \\
                        \depsVar_r
                        &= \braces{ q \in \prt(\gtc{G}) \mid \depsOn q r \gtc{G_1} }
                        \\
                        \RTs_r^1(q)
                        &= (r {?} s) {!} q^{\braces{\vect \ell_0}} \rtBraces{ i \rtc. (\gtc{H_i} \wrt (r,q)^{\braces{\vect \ell_0,i}}) }_{i \in I}
                        \quad [q \in \depsVar_r]
                        \\
                        \RTs_r^1(q)
                        &= \bigcup_{i \in I} (\gtc{H_i} \wrt (r,q)^{\braces{\vect \ell_0,i}})
                        \quad [q \in \prt(\gtc{G}) \setminus \braces{s,r} \setminus \depsVar_r]
                        \\
                        \mc{M_r^1}
                        &= r {?} s \mBraces{ \msg i<T_i> \mc. r {!} \depsVar_r (i) \mc. \gtToMon( \gtc{H_i} , r , D_r^0 ) }_{i \in I}
                    \end{align*}

                    Since $s$ sends $\msg j<T_j>$ to $r$, by Transition~\ruleLabel{out-mon-buf}, this message will end up in the buffer of the monitored blackbox of $r$.
                    The monitor $\mc{M_r^1}$ moves this message to the blackbox's buffer, and proceeds to send dependency messages $j$ to all $q \in \depsVar_r$ (concurrently).
                    Following the same reasoning as above, the monitored blackbox will keep outputting the dependencies above, possibly interleaved with $\tau$-transitions (from the blackbox or from the buffered blackbox reading messages), before doing anything else.

                    Let $P_r^2$ and $\vect m_r^2$ be the resulting blackbox and buffer, respectively.
                    By satisfaction (\Cref{d:mpstMon:satisfaction}), we have the following (applying \Cref{l:mpstMon:satisfactionDropLabel,l:mpstMon:satisfactionDropUnion,l:mpstMon:monitorDropInactive} immediately):
                    \begin{align*}
                        \RTs_r^2
                        &\deq \braces{ ( q , \gtc{H_j} \wrt (r,q)^{\braces{\vect \ell_0,j}} ) \mid q \in \prt(\gtc{G}) \setminus \braces{r} }
                        \\
                        D_r^2
                        &\deq \braces{ q \in \prt(\gtc{G}) \setminus \braces{s} \mid \unfold\big(\RTs_r^2(q)\big) \neq \tEnd }
                        \\
                        \mc{M_r^2}
                        &\deq \gtToMon( \gtc{H_j} , r , D_r^2 )
                        \\
                        &\satisfies{ \monImpl{ \bufImpl{ r }{ P_r^2 }{ \vect m_r^2 } }{ \mc{M_r^2} }{ \epsi } }{ \RTs_r^2 }[ r ]
                    \end{align*}
                    In conclusion,
                    \[
                        ( \gtc{H_j} , r , \braces{\vect \ell_0,j} , \prt(\gtc{G}) \setminus \braces{r} ) \initSatisfied ( \RTs_r^2 , D_r^2 , \mc{M_r^2} , P_r^2 , \vect m_r^2 ),
                    \]
                    such that the premise of \ih1 is satisfied for $r$.

                \item
                    For every $q \in \depsVar_s \setminus \depsVar_r$ we have:
                    \begin{align*}
                        \RTs_q^1(s)
                        &= (s {!} r) {!} q^{\braces{\vect \ell_0}} \rtBraces{ i \rtc. (\gtc{H_i} \wrt (q,s)^{\braces{\vect \ell_0}+i}) }_{i \in I}
                        \\
                        \RTs_q^1(q')
                        &= \bigcup_{i \in I} (\gtc{H_i} \wrt (q,q')^{\braces{\vect \ell_0}+i})
                        \quad [q' \in \prt(\gtc{G}) \setminus \braces{q,s}]
                        \\
                        \mc{M_q^1}
                        &= q {?} s \mBraces{ i \mc. \gtToMon( \gtc{H_i} , q , D_q^0 ) }_{i \in I}
                    \end{align*}

                    Since $s$ sends $j$ to $q$, by Transition~\ruleLabel{out-mon-buf}, this message will end up in the buffer of the monitored blackbox of $q$.
                    The monitor $\mc{M_q^1}$ moves this message to the blackbox's buffer.

                    Let $P_q^2$ and $\vect m_q^2$ be the resulting blackbox and buffer, respectively.
                    By satisfaction (\Cref{d:mpstMon:satisfaction}), we have the following (applying \Cref{l:mpstMon:satisfactionDropLabel,l:mpstMon:satisfactionDropUnion,l:mpstMon:monitorDropInactive} immediately):
                    \begin{align*}
                        \RTs_q^2
                        &\deq \braces{ ( q' , \gtc{H_j} \wrt (q,q')^{\braces{\vect \ell_0,j}} ) \mid q' \in \prt(\gtc{G}) \setminus \braces{q} }
                        \\
                        D_q^2
                        &\deq \braces{ q' \in \prt(\gtc{G}) \setminus \braces{q} \mid \unfold\big(\RTs_q^2(q')\big) \neq \tEnd }
                        \\
                        \mc{M_q^2}
                        &\deq \gtToMon( \gtc{H_j} , q , D_q^2 )
                        \\
                        &\satisfies{ \monImpl{ \bufImpl{ q }{ P_q^2 }{ \vect m_q^2 } }{ \mc{M_q^2} }{ \epsi } }{ \RTs_q^2 }[ q ]
                    \end{align*}
                    In conclusion,
                    \[
                        ( \gtc{H_j} , q , \braces{\vect \ell_0,j} , \prt(\gtc{G}) \setminus \braces{q} ) \initSatisfied ( \RTs_q^2 , D_q^2 , \mc{M_q^2} , P_q^2 , \vect m_q^2 ),
                    \]
                    such that the premise of \ih1 is satisfied for $q$.

                \item
                    For every $q \in \depsVar_r \setminus \depsVar_s$, the procedure is similar to above.

                \item
                    For every $q \in \depsVar_s \cap \depsVar_r$ we have:
                    \begin{align*}
                        \RTs_q^1(s)
                        &= (s {!} r) {!} q^{\braces{\vect \ell_0}} \rtBraces{ i \rtc. (\gtc{H_i} \wrt (q,s)^{\braces{\vect \ell_0}+i}) }_{i \in I}
                        \\
                        \RTs_q^1(r)
                        &= (r {?} s) {!} q^{\braces{\vect \ell_0}} \rtBraces{ i \rtc. (\gtc{H_i} \wrt (q,r)^{\braces{\vect \ell_0}+i}) }_{i \in I}
                        \\
                        \RTs_q^1(q')
                        &= \bigcup_{i \in I} (\gtc{H_i} \wrt (q,q')^{\braces{\vect \ell_0}+i})
                        \quad [q' \in \prt(\gtc{G}) \setminus \braces{q,s,r}]
                        \\
                        \mc{M_q^1}
                    &= q {?} s \mBraces[\big]{ i \mc. q {?} r \mBraces{ i \mc. \gtToMon( \gtc{H_i} , q , D_q^0 ) } \cup \mBraces{ j \mc. \pError }_{j \in I \setminus \braces{i}} }_{i \in I}
                    \end{align*}

                    By Transition~\ruleLabel{out-mon-buf}, from $s$ and $r$ the message $j$ will end up in the buffer of the monitored blackbox of $q$.
                    The monitor $\mc{M_q^1}$ first moves the message from $s$ to the blackbox's buffer, and then the message from $r$.
                    Even if the message from~$r$ is the first to end up in the monitor's buffer, this order of reception is enforced because buffers allow exchange of message from different senders.
                    Because $s$ and~$r$ send the same $j$, the monitor will not reach the error state.

                    Let $P_q^2$ and $\vect m_q^2$ be the resulting blackbox and buffer, respectively.
                    By satisfaction (\Cref{d:mpstMon:satisfaction}), we have the following (applying \Cref{l:mpstMon:satisfactionDropLabel,l:mpstMon:satisfactionDropUnion,l:mpstMon:monitorDropInactive} immediately):
                    \begin{align*}
                        \RTs_q^2
                        &\deq \braces{ ( q' , \gtc{H_j} \wrt (q,q')^{\braces{\vect \ell_0,j}} ) \mid q' \in \prt(\gtc{G}) \setminus \braces{q} }
                        \\
                        D_q^2
                        &\deq \braces{ q' \in \prt(\gtc{G}) \setminus \braces{q} \mid \unfold\big(\RTs_q^2(q')\big) \neq \tEnd }
                        \\
                        \mc{M_q^2}
                        &\deq \gtToMon( \gtc{H_j} , q , D_q^2 )
                        \\
                        &\satisfies{ \monImpl{ \bufImpl{ q }{ P_q^2 }{ \vect m_q^2 } }{ \mc{M_q^2} }{ \epsi } }{ \RTs_q^2 }[ q ]
                    \end{align*}
                    In conclusion,
                    \[
                        ( \gtc{H_j} , q , \braces{\vect \ell_0,j} , \prt(\gtc{G}) \setminus \braces{q} ) \initSatisfied ( \RTs_q^2 , D_q^2 , \mc{M_q^2} , P_q^2 , \vect m_q^2 ),
                    \]
                    such that the premise of \ih1 is satisfied for $q$.

                \item
                    For every $q \in \prt(\gtc{G}) \setminus \braces{s,r} \setminus \depsVar_s \setminus \depsVar_r$ we have:
                    \begin{align*}
                        \RTs_q^1(q')
                        &= \bigcup_{i \in I} (\gtc{H_i} \wrt (q,q')^{\braces{\vect \ell_0}+i})
                        \quad [q' \in \prt(\gtc{G}) \setminus \braces{q}]
                        \\
                        \mc{M_q^1}
                        &= \gtToMon( \gtc{H_k} , q , D_q^0 )
                        \quad [\text{arbitrary } k \in I]
                    \end{align*}
                    If we observe transitions from $q$, we are guaranteed that the behavior is independent of the current exchange between $s$ and $r$ in $\gtc{G_1}$: otherwise, $q \in \depsVar_s \cup \depsVar_r$.
                    Hence, as mentioned before, we can safely postpone these steps.

                    We have (applying \cref{l:mpstMon:satisfactionDropLabel,l:mpstMon:monitorDropInactive}):
                    \begin{align*}
                        \RTs_q^2
                        &\deq \braces{ ( q' , \gtc{H_j} \wrt (q,q')^{\braces{\vect \ell_0,j}} ) \mid q' \in \prt(\gtc{G}) \setminus \braces{q} }
                        \\
                        D_q^2
                        &\deq \braces{ q' \in \prt(\gtc{G}) \setminus \braces{q} \mid \unfold\big(\RTs_q^2(q')\big) \neq \tEnd }
                        \\
                        \mc{M_q^2}
                        &\deq \gtToMon( \gtc{H_j} , q , D_q^2 )
                        \\
                        &\satisfies{ \monImpl{ \bufImpl{ p }{ P_q^0 }{ \vect m_q^0 } }{ \mc{M_q^2} }{ \epsi } }{ \RTs_q^2 }[ q ]
                    \end{align*}
                    In conclusion,
                    \[
                        ( \gtc{H_j} , q , \braces{\vect \ell_0,j} , \prt(\gtc{G}) \setminus \braces{q} ) \initSatisfied ( \RTs_q^2 , D_q^2 , \mc{M_q^2} , P_q^2 , \vect m_q^2 ),
                    \]
                    such that the premise of \ih1 is satisfied for $q$.
            \end{itemize}

            At this point, the entire premise of \ih1 is satisfied.
            Hence, the thesis follows by \ih1.

        \item
            End: $\gtc{G_0} = \tEnd$.

            For every $p \neq q \in \prt(\gtc{G})$, we have $\RTs_p^1(q) = \tEnd$.
            By the definition of initial satisfaction (\Cref{d:mpstMon:initSatisfied}), for each $p \in \prt(\gtc{G})$, $\mc{M_p^1} \in \braces{\tEnd,\checkmark}$.
            By the progress properties of satisfaction (\Cref{d:mpstMon:satisfaction}), each monitored blackbox will continue performing transitions; these can only be $\tau$-transitions, for any other transition leads to an error signal or a violation of satisfaction.
            That is, each monitored blackbox will be reading messages from buffers or doing internal computations.
            However, there are only finitely many messages, and, by \bref{bTau}{Finite $\tau$}, each blackbox is assumed to only perform finitely many $\tau$-transitions in a row.
            Hence, at some point, we must see an $\tEnd$-transition from each monitored blackbox.

            Suppose the observed transition originates from $p \in \prt(\gtc{G})$.
            Suppose $\mc{M_p^1} = \tEnd$.
            Then the transition is either labeled $\tau$ or $\tEnd$.
            If the label is $\tau$, \ih1 applies with premise trivially satisfied.
            If the label is $\tEnd$, we follow \satref{sEnd}{End} to apply \ih1.

            It cannot be that $\mc{M_p^1} = \checkmark$: the LTS for networks (\Cref{d:mpstMon:ltsNetworks}) does not define any transitions, contradicting the observed transition.
            \qedhere
    \end{itemize}
\end{proof}

\begin{theorem}[Soundness]
    \label{t:mpstMon:soundness}
    \leavevmode
    \begin{itemize}
        \item Suppose given a well-formed global type $\gtc{G}$.
        \item For every $p \in \prt(\gtc{G})$, suppose given a process $P_p$, and take $\RTs_p,\mc{M_p}$ such that
            \[
                ( \gtc{G} , p , \braces{\epsi} , \prt(\gtc{G}) \setminus \braces{p} ) \initSatisfied ( \RTs_p , \prt(\gtc{G}) \setminus \braces{p} , \mc{M_p} , P_p , \epsi ).
            \]
        \item Suppose also $\prod_{p \in \prt(\gtc{G})} \monImpl{ \bufImpl{ p }{ P_p }{ \epsi } }{ \mc{M_p} }{ \epsi } \trans* \net P'$.
    \end{itemize}
    Then there exist $\vect \ell',\gtc{G'}$ such that
    \begin{itemize}
        \item $\gtc{G} \ltrans{ \vect \ell' } \gtc{G'}$,
        \item for every $p \in \prt(\gtc{G})$ there exist $P'_p,\vect m'_p,\RTs'_p,D'_p,\mc{M'_p}$ such that
            \begin{itemize}
                \item $( \gtc{G'} , p , \braces{ \ell' } , \prt(\gtc{G}) \setminus \braces{p} ) \initSatisfied ( \RTs'_p , D'_p , \mc{M'_p} , P'_p , \vect m'_p )$, and
                \item $\net P' \trans* \prod_{p \in \prt(\gtc{G})} \monImpl{ \bufImpl{ p }{ P'_p }{ \vect m'_p } }{ \mc{M'_p} }{ \epsi }$.
           \end{itemize}
   \end{itemize}
\end{theorem}

\begin{proof}
    Follows directly from \Cref{l:mpstMon:soundnessGeneralized}, given $\gtc{G} \ltrans{ \epsi } \gtc{G}$.
\end{proof}

\begin{theorem}[Error Freedom]
    \label{t:mpstMon:errorFreedom}
    Suppose given a well-formed global type $\gtc{G}$.
    For every $p \in \prt(G)$, suppose given a process $P_p$, and take $\RTs_p,\mc{M_p}$ such that
    \[
        ( \gtc{G} , p , \braces{ \epsi } , \prt(\gtc{G}) \setminus \braces{p} ) \initSatisfied ( \RTs_p , \prt(\gtc{G}) \setminus \braces{p} , \mc{M_p} , P_p , \epsi ).
    \]
    If $\net P \deq \prod_{p \in P} \monImpl{ \bufImpl{ p }{ P_p }{ \epsi } }{ \mc{M_p} }{ \epsi } \trans* \net P'$ , then there are no $\net Q,D$ such that $\net P' \trans* \net Q \| \pError_D$.
\end{theorem}

\begin{proof}
    Suppose, towards a contradiction, that $\net P' \trans* \net Q \| \pError_D$.
    Then, by multiple applications of Transition~\ruleLabel{error-par}, $\net Q \| \pError_D \trans* \pError_{D'}$ for some $D'$.
    Hence, $\net P \trans* \pError_{D'}$.
    Then, by soundness (\Cref{t:mpstMon:soundness}), $\pError_{D'}$ would further transition.
    However, the LTS of networks does not specify any transitions for error signals: a contradiction.
\end{proof}

\subsection{Proof of Transparency}
\label{s:transparencyProof}

Here we prove \Cref{t:mpstMon:transparencyBisim}, which is the full version of \Cref{t:mpstMon:transparencyMain} (\Cpageref{t:mpstMon:transparencyMain}).
We start with an overview of intermediate results used for the proof:
\begin{itemize}

    \item
        \Cref{d:mpstMon:ltsRelativeTypes} defines an LTS for relative types, where transitions can only be (dependency) outputs.

    \item
        \Cref{d:mpstMon:coherentSetup} defines a relation between a global type, participant, monitor, $\RTs$, and buffer.
        The idea is that the relative types in $\RTs$ may be in an intermediate state between exchanges in the global type, reflected by messages in the buffer.

    \item
        \Cref{l:mpstMon:satisfactionNoError} shows that monitored blackboxes in a satisfaction relation do not transition to an error state through the enhanced LTS (\Cref{d:mpstMon:enhancedLTS}).

    \item
        \Cref{l:mpstMon:satisfactionReverseBuffer} shows that messages in the buffer of a monitored blackbox of $p$ in a satisfaction relation are related to exchanges to $p$ in $\RTs$ of the satisfaction relation.

    \item
        \Cref{l:mpstMon:satisfactionOutput} shows that an output transitions of a monitored blackbox in a satisfaction relation implies that the monitor is a related output, possibly preceded by a dependency output.

    \item
        \Cref{l:mpstMon:satisfactionEnd} shows that an $\tEnd$-transition of a monitored blackbox in a satisfaction relation implies that the monitor is $\tEnd$, possibly preceded by a dependency output.

    \item
        \Cref{l:mpstMon:finiteReceive} shows that a message in the buffer of a monitored blackbox relates to an exchange in a global type, reachable in finitely many steps.

    \item
        \Cref{l:mpstMon:monitorNoOutput} shows that if the blackbox of a monitored blackbox in a satisfaction relation does an input transition, then the monitor is a sequence of inputs ending in an input related to the input transition.

    \item
        \Cref{l:mpstMon:monitorNotCheck} shows that if the blackbox of a monitored blackbox in a satisfaction relation does a $\tau$-transition, then the monitor is not $\mc{\checkmark}$.

    \item
        \Cref{l:mpstMon:labelOracleDepOut} shows that label oracles are equal for relative types prior and past dependency outputs.

    \item
        \Cref{l:mpstMon:LOAlpha} shows a precise relation between actions and label oracles.

    \item
        \Cref{t:mpstMon:transparencyBisim} shows transparency.

\end{itemize}

\begin{definition}[LTS for Relative Types]
    \label{d:mpstMon:ltsRelativeTypes}
    We define an LTS for relative types, denoted $\rtc{R} \ltrans{ m }_p \rtc{R'}$, where actions are output (dependency) messages $m$ (\Cref{f:mpstMon:netGrammars}~(top)) and $p$ is the receiving participant, by the following rules:
    \begin{mathpar}
        \begin{bussproof}[rel-exch]
            \bussAssume{
                j \in I
            }
            \bussUn{
                q {!} p^{\mbb L} \rtBraces{ \msg i<T_i> \rtc. \rtc{R_i} }_{i \in I}
                \mathrel{\raisebox{-2.0pt}{$\ltrans{ q {!} p (\msg j<T_j>) }_p$}}
                \rtc{R_j}
            }
        \end{bussproof}
        \and
        \begin{bussproof}[rel-dep]
            \bussAssume{
                j \in I
            }
            \bussUn{
                (q \lozenge r) {!} p^{\mbb L} \rtBraces{ i \rtc. \rtc{R_i} }_{i \in I}
                \mathrel{\raisebox{-2.0pt}{$\ltrans{ q {!} p \Parens{j} }_p$}}
                \rtc{R_j}
            }
        \end{bussproof}
        \and
        \begin{bussproof}[rel-rec]
            \bussAssume{
                \unfold( \mu X^{\mbb L} \rtc. \rtc{R} ) \ltrans{m}_p \rtc{R'}
            }
            \bussUn{
                \mu X^{\mbb L} \rtc. \rtc{R} \ltrans{m}_p \rtc{R'}
            }
        \end{bussproof}
    \end{mathpar}
    Given $\vect m = m_1 , \ldots , m_k$, we write $\rtc{R} \ltrans{ \vect m }_p \rtc{R'}$ to denote $\rtc{R} \ltrans{ m_1 }_p \ldots \ltrans{m_k}_p \rtc{R'}$ (reflexive if $\vect m = \epsi$).
    We write $\vect m(q)$ to denote the subsequence of messages from $\vect m$ sent by $q$.
\end{definition}

\begin{definition}[Coherent Setup ($\coh$)]
    \label{d:mpstMon:coherentSetup}
    A global type $\gtc{G_0}$, a participant $p$, a monitor \mbox{$\mc{M} \neq \mu X \mc. \mc{M'}$}, a map $\RTs : \bm{P} \rightarrow \bm{R}$, and a sequence of messages $\vect n$ are \emph{in a coherent setup}, denoted $\coherent{\gtc{G_0}}{p}{\mc{M}}{\RTs}{\vect n}$, if and only if there exist $\gtc{G},\vect \ell$ such that $\gtc{G_0} \ltrans{ \vect \ell } \gtc{G}$, $\vect n$ exclusively contains messages with sender in $\dom(\RTs)$, and
    \begin{itemize}
        \item
            \textbf{(Initial state)}
            $\mc{M} \in \braces{ \gtToMon( \gtc{G} , p , \prt(\gtc{G_0}) \setminus \braces{p} ) , \checkmark }$ implies that
            \begin{itemize}
                \item for every $q \in \dom(\RTs)$ there exists $\mbb L_q$ such that $(\gtc{G} \wrt (p,q)^{\mbb L_q}) \ltrans{ \vect n(q) }_p \RTs(q)$, and
                \item if $\mc{M} = \checkmark$ then $\gtToMon( \gtc{G} , p , \prt(\gtc{G_0}) \setminus \braces{p} ) = \tEnd$;
            \end{itemize}

        \item
            \textbf{(Intermediate state)}
            Otherwise, there exist $\gtc{G'},j$ such that $\gtc{G} \ltrans{j} \gtc{G'}$, and either of the following holds:
            \begin{itemize}
                \item all of the following hold:
                    \begin{itemize}
                        \item $\mc{M} = p {!} D (j) \mc. \gtToMon( \gtc{G'} , p , \prt(\gtc{G_0}) \setminus \braces{p} )$,
                        \item for every $q \in \dom(\RTs) \setminus D$ there exists $\mbb L_q$ such that $(\gtc{G'} \wrt (p,q)^{\mbb L_q}) \ltrans{ \vect n(q) }_p \RTs(q)$, and
                        \item for every $q \in D$ there exists $\mbb L_q$ such that $\RTs(q) = \gtc{G} \wrt (p,q)^{\mbb L_q}$;
                    \end{itemize}

                \item all of the following hold:
                    \begin{itemize}
                        \item $\mc{M} = p {?} r \mBraces{ j \mc. \gtToMon( \gtc{G'} , p , \prt(\gtc{G_0}) \setminus \braces{p} ) } \cup \mBraces{ i \mc. \pError }_{ i \in I \setminus \braces{j} }$,
                        \item for every $q \in \dom(\RTs) \setminus \braces{r}$ there exists $\mbb L_q$ such that $(\gtc{G'} \wrt (p,q)^{\mbb L_q}) \ltrans{ \vect n(q) }_p \RTs(q)$, and
                        \item there exists $\mbb L_r$ such that $(\gtc{G} \wrt (p,r)^{\mbb L_r}) \ltrans{ \vect n(r) }_p \RTs(r)$.
                    \end{itemize}
            \end{itemize}
    \end{itemize}
\end{definition}

\begin{lemma}
    \label{l:mpstMon:satisfactionNoError}
    Suppose
    \begin{itemize}
        \item $\mcl{R} \satisfies{ \monImpl{ \bufImpl{ p }{ P_0 }{ \epsi } }{ \mc{M_0} }{ \epsi } }{ \RTs_0 }[ p ]$,
        \item where $\RTs_0 \deq \braces{ ( q , \gtc{G_0} \wrt (p,q)^{\braces{\epsi}} ) \mid q \in \prt(\gtc{G_0}) \setminus \braces{p} }$ for well-formed $\gtc{G_0}$.
    \end{itemize}
    Then
    \begin{itemize}
        \item for any $( \monImpl{ \bufImpl{ p }{ P }{ \vect m } }{ \mc{M} }{ \vect n } , \RTs , \Lbls ) \in \mcl{R}$ such that $\coherent{ \gtc{G_0} }{ p }{ \mc{M} }{ \RTs }{ \vect n }$, and
        \item for any $\alpha,\net P',\Omega,\Omega'$ such that $\monImpl{ \bufImpl{ p }{ P }{ \vect m } }{ {\mc{M}} }{ \vect n } \bLtrans{\Omega}{\alpha}{\Omega'} \net P'$,
    \end{itemize}
    we have $\net P' \not\equiv \pError_{\braces{p}}$.
\end{lemma}

\begin{proof}
    Suppose, towards a contradiction, that $\net P' \equiv \pError_{\braces{p}}$.
    Then the transition is due to Transition~\ruleLabel{no-dep}, $\alpha = \tau$, $\Omega' = \Omega$, and the transition is derived from Transition~\ruleLabel{error-out}, \ruleLabel{error-in}, or~\ruleLabel{error-mon}.
    We show in each case separately that this leads to a contradiction.
    \begin{itemize}
        \item
            Transition~\ruleLabel{error-out}.
            Then $\mc{M}$ denotes an output, and so, by the definition of coherent setup (\Cref{d:mpstMon:coherentSetup}), there is a $q \in \dom(\RTs)$ such that $\unfold\big(\RTs(q)\big)$ is an output with a location that prefixes all other locations in $\RTs$.
            By \satref{sTau}{Tau}, $( \pError_{\braces{p}} , \RTs , \Lbls ) \in \mcl{R}$.
            Then by the progress property of satisfaction (\Cref{d:mpstMon:satisfaction}), there exist $\alpha,\net P''$ such that $\pError_{\braces{p}} \ltrans{\alpha} \net P''$.
            However, the LTS for networks (\Cref{d:mpstMon:ltsNetworks}) does not define such a transition: a contradiction.

        \item
            Transition~\ruleLabel{error-in}.
            Then the monitor tries to read a message from $\vect n$, but the first relevant message is incorrect.
            By the definition of coherent setup (\Cref{d:mpstMon:coherentSetup}), all messages in $\vect n$ are in accordance with \satref{sInput}{Input} or~\satclause{sDependencyInput}{Dependency input}.
            But then the message cannot be incorrect: a contradiction.

        \item
            Transition~\ruleLabel{error-mon}.
            This means that $\mc{M} = \pError$.
            This can only be the case after two related dependency inputs of different labels.
            However, the messages read must be due to \satref{sDependencyInput}{Dependency input}, where the first one records the chosen label in~$\Lbls$, and the second one uses that label.
            Hence, it is not possible that two different labels have been received: a contradiction.
    \end{itemize}
    Hence, $\net P' \not\equiv \pError_{\braces{p}}$.
\end{proof}

\begin{lemma}
    \label{l:mpstMon:satisfactionReverseBuffer}
    Suppose
    \begin{itemize}
        \item $\mcl{R} \satisfies{ \monImpl{ \bufImpl{ p }{ P_0 }{ \epsi } }{ \mc{M_0} }{ \epsi } }{ \RTs_0 }[ p ]$,
        \item where $\RTs_0 \deq \braces{ ( q , \gtc{G_0} \wrt (p,q)^{\braces{\epsi}} ) \mid q \in \prt(\gtc{G_0}) \setminus \braces{p} }$
        \item for well-formed $\gtc{G_0}$ with $p \in \prt(\gtc{G_0})$.
    \end{itemize}
    Then
    \begin{itemize}
        \item for any $( \monImpl{ \bufImpl{ p }{ P }{ \vect m } }{ \mc{M} }{ \vect n } , \RTs , \Lbls ) \in \mcl{R}$
        \item with non-empty $\vect n$
        \item such that $\coherent{ \gtc{G_0} }{ p }{ \mc{M} }{ \RTs }{ \vect n }$,
    \end{itemize}
    there exist
    \begin{itemize}
        \item $\RTs',\Lbls'$ such that $( \monImpl{ \bufImpl{ p }{ P }{ \vect m } }{ \mc{M} }{ \epsi } , \RTs' , \Lbls' ) \in \mcl{R}$, and
        \item $(q,\rtc{R}) \in \dom(\RTs')$ such that
            \begin{itemize}
                \item $\unfold(\rtc{R}) = q {!} p^{\mbb L} \rtBraces{ \msg i<T_i> \rtc. \rtc{R_i} }_{i \in I}$
                \item or $\unfold(\rtc{R}) = (q \lozenge r) {!} p^{\mbb L} \rtBraces{ i \rtc. \rtc{R_i} }_{i \in I}$.
            \end{itemize}
    \end{itemize}
\end{lemma}

\begin{proof}
    Take any $( \monImpl{ \bufImpl{ p }{ P }{ \vect m } }{ \mc{M} }{ \vect n } , \RTs , \Lbls ) \in \mcl{R}$ with non-empty $\vect n$ such that
    \[
        \coherent{ \gtc{G_0} }{ p }{ \mc{M} }{ \RTs }{ \vect n }.
    \]
    By the minimality of $\mcl{R}$, each message in $\vect n$ traced back to applications of \satref{sInput}{Input} or~\satclause{sDependencyInput}{Dependency input}.
    By the definition of coherent setup (\Cref{d:mpstMon:coherentSetup}), each relative type in $\RTs$ relates $\gtc{G_0}$ with the messages in $\vect n$ through the LTS for relative types (\Cref{d:mpstMon:ltsRelativeTypes}).
    Hence, at $( \monImpl{ \bufImpl{ p }{ P }{ \vect m } }{ \mc{M} }{ \epsi } , \RTs' , \Lbls' ) \in \mcl{R}$, since $\vect n$ is non-empty, the unfolding of at least one relative type in $\RTs'$ denotes an input by $p$.
\end{proof}

\begin{lemma}
    \label{l:mpstMon:satisfactionOutput}
    \leavevmode
    \begin{itemize}
        \item Suppose
            \begin{itemize}
                \item $\mcl{R} \satisfies{ \monImpl{ \bufImpl{ p }{ P_0 }{ \epsi } }{ \mc{M_0} }{ \epsi } }{ \RTs_0 }[ p ]$,
                \item where $\RTs_0 \deq \braces{ ( q , \gtc{G_0} \wrt (p,q)^{\braces{\epsi}} ) \mid q \in \prt(\gtc{G_0}) \setminus \braces{p} }$
                \item for well-formed $\gtc{G_0}$ with $p \in \prt(\gtc{G_0})$.
            \end{itemize}
        \item Take any $( \monImpl{ \bufImpl{ p }{ P }{ \vect m } }{ \mc{M} }{ \vect n } , \RTs , \Lbls ) \in \mcl{R}$ such that $\coherent{ \gtc{G_0} }{ p }{ \mc{M} }{ \RTs }{ \vect n }$.
        \item Suppose $\bufImpl{ p }{ P }{ \vect m } \ltrans{ p {!} q (\msg j<T_j>) } \bufImpl{ p }{ P' }{ \vect m }$.
    \end{itemize}
    Then
    \begin{itemize}
        \item $\unfold(\mc{M}) = \mc{M'}$
        \item or $\unfold(\mc{M}) = p {!} D (\ell) \mc. \mc{M'}$
    \end{itemize}
    where $\mc{M'} = p {!} q \mBraces{ \msg i<T_i> \mc. \mc{M'_i} }_{i \in I}$ s.t.\ $j \in I$.
\end{lemma}

\begin{proof}
    Suppose, toward a contradiction, that $\unfold(\mc{M}) \neq \mc{M'}$ and $\unfold(\mc{M}) \neq p {!} D (\ell) \mc. \mc{M'}$.
    Then, by Transition~\ruleLabel{error-out}, $\monImpl{ \bufImpl{ p }{ P }{ \vect m } }{ \mc{M} }{ \vect n } \ltrans{\tau} \pError_{\braces{p}}$.
    By \satref{sTau}{Tau}, $( \pError_{\braces{p}} , \RTs , \Lbls ) \in \mcl{R}$.
    If $\mc{M}$ starts with a sequence of recursive definitions, we unfold them; by Transition~\ruleLabel{mon-rec}, the behavior is unchanged; let us assume, w.l.o.g., that $\mc{M}$ does not start with a recursive definition.
    From each possible shape of $\mc{M}$, we derive a contradiction.
    \begin{itemize}
        \item
            $\mc{M} = p {!} r \mBraces{ \msg k<T_k> \mc. \mc{M_k} }_{k \in K}$ for $r \neq q$.
            By definition of coherent setup (\Cref{d:mpstMon:coherentSetup}), $\gtc{G_0} \ltrans{\vect \ell} \gtc{G}$, $\mc{M} = \gtToMon( \gtc{G} , p , \prt(\gtc{G_0}) \setminus \braces{p} )$, and there exists $\mbb L$ such that \mbox{$(\gtc{G} \wrt (p,r)^{\mbb L}) \ltrans{ \vect n(r) }_p \RTs(r)$}.
            Then $\unfold( \gtc{G} \wrt (p,r)^{\mbb L} ) = p {!} r \rtBraces{ \msg k<T_k> \rtc. \rtc{R_k} }_{k \in K}$.
            Since the LTS for relative types (\Cref{d:mpstMon:ltsRelativeTypes}) only allows input transitions, then $\RTs(r) = \gtc{G} \wrt (p,r)^{\mbb L}$.
            Clearly, $\mbb L$ prefixes all other locations in $\RTs$.
            Hence, by the progress property of satisfaction (\Cref{d:mpstMon:satisfaction}), there exists $\net P'$ such that $\pError_{\braces{p}} \ltrans{ \tau } \net P'$.
            However, the LTS for networks (\Cref{d:mpstMon:ltsNetworks}) does not define such a transition: a contradiction.

        \item
            $\mc{M} = p {!} q \mBraces{ \msg i<T_i> \mc. \mc{M_i} }_{i \in I}$ for $j \notin I$.
            Analogous to the case above.

        \item
            $\mc{M} = p {?} r \mBraces{ \msg k<T_k> \mc. \mc{M_k} }_{k \in K}$.
            By definition of coherent setup (\Cref{d:mpstMon:coherentSetup}), \mbox{$\gtc{G_0} \ltrans{\vect \ell} \gtc{G}$}, $\mc{M} = \gtToMon( \gtc{G} , p , \prt(\gtc{G_0}) \setminus \braces{p} )$, and there exists $\mbb L$ such that \mbox{$(\gtc{G} \wrt (p,r)^{\mbb L}) \ltrans{ \vect n(r) }_p \RTs(r)$}.
            Whether $\vect n$ is empty or not, there exist $\RTs',\Lbls'$ such that \mbox{$( \monImpl{ \bufImpl{ p }{ P }{ \vect m } }{ \mc{M} }{ \epsi } , \RTs' , \Lbls' ) \in \mcl{R}$}: if empty, this holds vacuously with $\RTs' \deq \RTs,\Lbls' \deq \Lbls$; if non-empty, it follows from \Cref{l:mpstMon:satisfactionReverseBuffer}.
            Then \mbox{$\unfold\big(\RTs'(r)\big) = \unfold( \gtc{G} \wrt (p,r)^{\mbb L} ) = r {!} p \rtBraces{ \msg k<T_k> \rtc. \rtc{R_k} }_{k \in K}$}.
            Again, by Transition~\ruleLabel{error-out}, $\monImpl{ \bufImpl{ p }{ P }{ \vect m } }{ \mc{M} }{ \epsi } \ltrans{\tau} \pError_{\braces{p}}$, and, by \satref{sTau}{Tau}, $( \pError_{\braces{p}} , \RTs' , \Lbls' ) \in \mcl{R}$.
            Then, by \satref{sInput}{Input}, $\pError_{\braces{p}} = \monImpl{ \bufImpl{ p }{ P }{ \vect m } }{ \mc{M'} }{ \vect n }$: a contradiction.

        \item
            $\mc{M} = p {?} r \mBraces{ k \mc. \mc{M_k} }_{k \in K}$.
            Analogous to the case above.

        \item
            $\mc{M} = \tEnd$.
            By definition of coherent setup (\Cref{d:mpstMon:coherentSetup}), $\gtc{G_0} \ltrans{\vect \ell} \gtc{G}$, \mbox{$\mc{M} = \gtToMon( \gtc{G} , p , \prt(\gtc{G_0}) \setminus \braces{p} )$}, and for every $q \in \dom(\RTs)$ there exists $\mbb L_q$ such that $(\gtc{G} \wrt (p,q)^{\mbb L_q}) \ltrans{ \vect n(q) }_p \RTs(q)$.
            Since the LTS for relative types (\Cref{d:mpstMon:ltsRelativeTypes}) only allows input transitions, for every $q \in \dom(\RTs)$, $\unfold\big(\RTs(q)\big) = \unfold( \gtc{G} \wrt (p,q)^{\mbb L_q} ) = \tEnd$.
            Then, by \satref{sEnd}{End}, there exists $\net P'$ such that $\pError_{\braces{p}} \trans* \ltrans{\tEnd} \net P'$.
            However, the LTS for networks (\Cref{d:mpstMon:ltsNetworks}) does not define such transitions: a contradiction.

        \item
            $\mc{M} = \pError$.
            By definition of coherent setup (\Cref{d:mpstMon:coherentSetup}), $\mc{M} \neq \pError$: a contradiction.

        \item
            $\mc{M} = \checkmark$.
            It cannot be the case that $\mc{M_0} = \checkmark$, and, by the LTS for networks (\Cref{d:mpstMon:ltsNetworks}), $\mc{M} = \checkmark$ can only be reached through a transition labeled $\tEnd$.
            Hence, \mbox{$\monImpl{ \bufImpl{ p }{ P }{ \vect m } }{ \checkmark }{ \vect n }$} must have been reached through \satref{sEnd}{End}.
            Then
            \[
                \monImpl{ \bufImpl{ p }{ P }{ \vect m } }{ \checkmark }{ \vect n } \ntrans:
            \]
            a contradiction.
    \end{itemize}
    Hence, the thesis must indeed hold.
\end{proof}

\begin{lemma}
    \label{l:mpstMon:satisfactionEnd}
    \leavevmode
    \begin{itemize}
        \item Suppose
            \begin{itemize}
                \item $\mcl{R} \satisfies{ \monImpl{ \bufImpl{ p }{ P_0 }{ \epsi } }{ \mc{M_0} }{ \epsi } }{ \RTs_0 }[ p ]$,
                \item where $\RTs_0 \deq \braces{ ( q , \gtc{G_0} \wrt (p,q)^{\braces{\epsi}} ) \mid q \in \prt(\gtc{G_0}) \setminus \braces{p} }$
                \item for well-formed $\gtc{G_0}$ with $p \in \prt(\gtc{G_0})$.
            \end{itemize}
        \item Take any $( \monImpl{ \bufImpl{ p }{ P }{ \vect m } }{ \mc{M} }{ \vect n } , \RTs , \Lbls ) \in \mcl{R}$ such that $\coherent{ \gtc{G_0} }{ p }{ \mc{M} }{ \RTs }{ \vect n }$.
        \item Suppose $\bufImpl{ p }{ P }{ \vect m } \ltrans{ \tEnd } \bufImpl{ p }{ P' }{ \vect m }$.
    \end{itemize}
    Then $\vect n = \epsi$, and
    \begin{itemize}
        \item $\unfold(\mc{M}) = \tEnd$
        \item or $\unfold(\mc{M}) = p {!} D (\ell) \mc. \tEnd$.
    \end{itemize}
\end{lemma}

\begin{proof}
    Suppose toward a contradiction that $\vect n \neq \epsi$ or that $\unfold(\mc{M}) \notin \braces{ \tEnd , p {!} D (\ell) \mc. \tEnd }$.
    Then, by Transition~\ruleLabel{error-end}, $\monImpl{ \bufImpl{ p }{ P }{ \vect m } }{ \mc{M} }{ \vect n } \ltrans{\tau} \pError_{\braces{p}}$.
    The contradiction follows similar to the reasoning in the proof of \Cref{l:mpstMon:satisfactionOutput}.
    We only discuss the additional case where $\vect n \neq \epsi$ and $\mc{M} = p {!} D (\ell) \mc. \mc{M'}$.
    By induction on the size of $D$, $\monImpl{ \bufImpl{ p }{ P }{ \vect m } }{ \mc{M} }{ \vect n } \trans* \monImpl{ \bufImpl{ p }{ P }{ \vect m } }{ \mc{M'} }{ \vect n }$.
    By the definition of coherent setup (\Cref{d:mpstMon:coherentSetup}), $\mc{M}$ is synthesized from a global type, so $\mc{M'}$ is not a dependency output (there are never two consecutive dependency outputs).
    Hence, the other cases for $\mc{M'}$ apply and the search for a contradiction continues as usual.
\end{proof}

\begin{lemma}
    \label{l:mpstMon:finiteReceive}
    Suppose
    \begin{itemize}
        \item $\coherent{ \gtc{G_0} }{ p }{ \mc{M} }{ \RTs }{ \vect n , u }$
        \item with $\gtc{G_0} \ltrans{ \vect \ell } \gtc{G} \ltrans{ k } \gtc{G'}$
        \item and $u = q {!} p (\msg j<T_j>)$.
    \end{itemize}
    Then
    \begin{itemize}
        \item $\gtc{G'} = q {!} p \gtBraces{ \msg i<T_i> \gtc. \gtc{G'_i} }_{i \in I}$ with $j \in I$
        \item or, for every $k',\gtc{G''}$ such that $\gtc{G'} \ltrans{ k' } \gtc{G''}$, there exists (finite) $\vect \ell'$ such that $\gtc{G''} \ltrans{ \vect \ell' } q {!} p \gtBraces{ \msg i<T_i> \gtc. \gtc{G''_i} }_{i \in I}$ with $j \in I$.
    \end{itemize}
\end{lemma}

\begin{proof}
    By the definition of coherent setup (\Cref{d:mpstMon:coherentSetup}), there are $\mbb L,\mbb L'$ such that $\unfold( \gtc{G'} \wrt (p,q)^{\mbb L} ) = q {!} p^{\mbb L'} \rtBraces{ \msg i<T_i> \rtc. \rtc{R_i} }_{i \in I}$ with $j \in I$.
    By definition, this relative type is generated in finitely many steps by \Cref{alg:mpstMon:relativeProjectionWithLocs} (\Cref{li:mpstMon:rpIndepLab}).
    We apply induction on the number $x$ of steps this took.
    In the base case, the thesis holds trivially.

    In the inductive case, the relative projections of all branches of $\gtc{G'}$ onto $(p,q)$ are the same.
    Take any $k',\gtc{G''}$ such that $\gtc{G'} \ltrans{k'} \gtc{G''}$.
    Then $\gtc{G''} \wrt (p,q)^{\mbb L+k'}$ is generated in less than $x$ steps through \Cref{alg:mpstMon:relativeProjectionWithLocs} (\cref{li:mpstMon:rpIndepLab}).
    Hence, the thesis follows from the IH.
\end{proof}

\begin{lemma}
    \label{l:mpstMon:monitorNoOutput}
    \leavevmode
    \begin{itemize}
        \item Suppose
            \begin{itemize}
                \item $\mcl{R} \satisfies{ \monImpl{ \bufImpl{ p }{ P_0 }{ \epsi } }{ \mc{M_0} }{ \epsi } }{ \RTs_0 }[ p ]$,
                \item where $\coherent{ \gtc{G_0} }{ p }{ \mc{M_0} }{ \RTs_0 }{ \epsi }$
                \item for well-formed $\gtc{G_0}$ with $p \in \prt(\gtc{G_0})$.
            \end{itemize}
        \item Take any $\RTs,\Lbls,\vect m,\vect n,\mc{M},P$ such that
            \begin{itemize}
                \item $( \monImpl{ \bufImpl{ p }{ P }{ \vect m } }{ \mc{M} }{ \vect n } , \RTs , \Lbls ) \in \mcl{R}$,
                \item $\coherent{ \gtc{G_0} }{ p }{ \mc{M} }{ \RTs }{ \vect n }$,
                \item and $\vect n$ is non-empty.
            \end{itemize}
    \end{itemize}
    If
    \begin{itemize}
        \item $P \ltrans{ p {?} q (\msg j<T_j>) } P'$
        \item and $\unfold(\mc{M}) \neq p {?} q \mBraces{ \msg i<T_i> \mc. \mc{M_i} }_{i \in I}$ for $I \supseteq \braces{j}$,
    \end{itemize}
    then $\unfold(\mc{M}) \in \braces{ p {?} r \mBraces{ \msg k<T_k> \mc. \mc{M_k} }_{k \in K} , p {?} r \mBraces{ k \mc. \mc{M_k} }_{k \in K} , p {!} D (\ell) \mc. \mc{M'} \mid r \neq q }$.
\end{lemma}

\begin{proof}
    Assume, toward a contradiction, that
    \[
        \unfold(\mc{M}) \notin \braces{ p {?} r \mBraces{ \msg k<T_k> \mc. \mc{M_k} }_{k \in K} , p {?} r \mBraces{ k \mc. \mc{M_k} }_{k \in K} , p {!} D (\ell) \mc. \mc{M'} \mid r \neq q }.
    \]
    We discuss each possible shape of $\unfold(\mc{M})$ separetely (w.l.o.g., assume $\mc{M}$ does not start with a recursive definition, i.e., $\unfold(\mc{M}) = \mc{M}$).
    \begin{itemize}
        \item
            $\mc{M} = p {!} r \mBraces{ \msg k<T_k> \mc. \mc{M_k} }_{k \in K}$.
            By definition of coherent setup (\Cref{d:mpstMon:coherentSetup}), $\unfold\big(\RTs(r)\big) = p {!} r^{\mbb L} \rtBraces{ \msg k<T_k> \rtc. \rtc{R_k} }_{k \in K}$ for some $\mbb L$.
            It must then be that $\mbb L$ is the earliest location in all $\RTs$.
            Then, by the progress property of satisfaction (\Cref{d:mpstMon:satisfaction}), the monitored blackbox must eventually do an output transition.
            However, by \bref{bInputOutput}{Input/Output}, since $P$ does an input transition, it cannot do an output transition: a contradiction.

        \item
            $\mc{M} = p {?} q \mBraces{ \msg i<T_i> \mc. \mc{M_i} }_{i \in I}$ for $j \notin I$.
            Analogous to the similar case in the proof of \Cref{l:mpstMon:satisfactionOutput}.

        \item
            $\mc{M} = \tEnd$.
            Analogous to the similar case in the proof of \Cref{l:mpstMon:satisfactionOutput}.

        \item
            $\mc{M} = \pError$.
            Analogous to the similar case in the proof of \Cref{l:mpstMon:satisfactionOutput}.

        \item
            $\mc{M} = \checkmark$.
            It cannot be the case the $\mc{M_0} = \checkmark$, and, by the LTS for networks (\Cref{d:mpstMon:ltsNetworks}), $\mc{M} = \checkmark$ can only be reached through a transition labeled $\tEnd$.
            Hence, the current monitored blackbox must have been reached through \satref{sEnd}{End} from $\monImpl{ \bufImpl{ p }{ P }{ \vect m } }{ \tEnd }{ \vect n}$.
            An $\tEnd$-transition from this state requires that $\vect n$ is empty, which it is not: a contradiction.
    \end{itemize}
    Hence, the thesis must hold indeed.
\end{proof}

\begin{lemma}
    \label{l:mpstMon:monitorNotCheck}
    \leavevmode
    \begin{itemize}
        \item Suppose
            \begin{itemize}
                \item $\mcl{R} \satisfies{ \monImpl{ \bufImpl{ p }{ P_0 }{ \epsi } }{ \mc{M_0} }{ \epsi } }{ \RTs_0 }[ p ]$,
                \item where $\coherent{ \gtc{G_0} }{ p }{ \mc{M_0} }{ \RTs_0 }{ \epsi }$
                \item for well-formed $\gtc{G_0}$ with $p \in \prt(\gtc{G_0})$.
            \end{itemize}
        \item Take any $\RTs,\Lbls,\vect m,\vect n,\mc{M},P$ such that
            \begin{itemize}
                \item $( \monImpl{ \bufImpl{ p }{ P }{ \vect m } }{ \mc{M} }{ \vect n } , \RTs , \Lbls ) \in \mcl{R}$, $\coherent{ \gtc{G_0} }{ p }{ M }{ \RTs }{ \vect n }$,
                \item and $\vect n$ is non-empty.
            \end{itemize}
    \end{itemize}
    If $P \ltrans{ \tau } P'$, then $\mc{M} \neq \checkmark$.
\end{lemma}

\begin{proof}
    Suppose, toward a contradiction, that $\mc{M} = \checkmark$.
    We have $\mc{M_0} \neq \checkmark$, so this state must have been reached through an $\tEnd$-transition between $P_0$ and $P$.
    However, by \bref{bEnd}{End}, there can be no transitions after an $\tEnd$-transition: a contradiction.
\end{proof}

\begin{lemma}
    \label{l:mpstMon:labelOracleDepOut}
    Suppose given $\RTs : \bm{P} \rightarrow \bm{R}$ and $\Lbls : \Pow(\vect {\bm{L}}) \rightarrow \bm{L}$.

    If there is $( q , \rtc{R} ) \in \RTs$ such that $\unfold(\rtc{R}) = (p \lozenge r) {!} q^{\mbb L} \rtBraces{ i \rtc. \rtc{R_i} }_{i \in I}$, then, for every $j \in I$,
    \[
        \LO( p , \RTs , \Lbls ) = \LO( p , \RTs \update{q \mapsto \rtc{R_j}} , \Lbls ).
    \]
\end{lemma}

\begin{proof}
    Trivally, by the definition of label oracle (\Cref{d:mpstMon:labelOracle}), the right-hand-side is a subset of the left-hand-side.
    For the other direction, the update of $\RTs$ might add additional sequences.
    However, since the update to $\RTs$ does not affect any other relative types and $\Lbls$ is not updated, no sequences are removed.
    Hence, the left-hand-side is a subset of the right-hand-side.
\end{proof}

\begin{lemma}
    \label{l:mpstMon:LOAlpha}
    Suppose given $\RTs : \bm{P} \rightarrow \bm{R}$, $\Lbls : \Pow(\vect {\bm{L}}) \rightarrow \bm{L}$, $\Omega = \LO(p,\RTs,\Lbls)$, and $\Omega' = \Omega(\alpha)$.

    Then all of the following hold:
    \begin{itemize}

        \item
            If $\alpha = p {!} q (\msg j<T_j>)$, then $\RTs(q) \unfoldeq p {!} q^{\mbb L} \rtBraces{ \msg i<T_i> \rtc. \rtc{R_i} }_{i \in I}$ with $j \in I$, and
            \[
                \Omega' = \LO(p,\RTs\update{q\mapsto \rtc{R_j}},\Lbls\update{\mbb L\mapsto j}).
            \]

        \item
            If $\alpha = p {?} q (\msg j<T_j>)$, then $\RTs(q) \unfoldeq q {!} p^{\mbb L} \rtBraces{ \msg i<T_i> \rtc. \rtc{R_i} }_{i \in I}$ with $j \in I$, and
            \[
                \Omega' = \LO(p,\RTs\update{q\mapsto \rtc{R_j}},\Lbls\update{\mbb L\mapsto j}).
            \]

        \item
            If $\alpha = p {?} q \Parens{j}$, then $\RTs(q) \unfoldeq (q \lozenge r) {!} p^{\mbb L} \rtBraces{ i \rtc. \rtc{R_i} }_{i \in I}$ with $j \in I$.
            If $\not\exists \mbb L' \in \dom(\Lbls).~\mbb L' \overlap \mbb L$, then $\Omega' = \LO(p,\RTs\update{q\mapsto \rtc{R_j}},\Lbls\update{\mbb L\mapsto j})$.
            If $\exists (\mbb L',j) \in \Lbls.~ \mbb L' \overlap \mbb L$, then
            \[
                \Omega' = \LO(p,\RTs\update{q\mapsto \rtc{R_j}},\Lbls).
            \]

        \item
            If $\alpha = \tEnd$, then $\forall (q,\rtc{R}) \in \RTs.~\rtc{R}\unfoldeq\tEnd$, and $\Omega' = \LO(p,\emptyset,\emptyset)$.

        \item
            If $\alpha = \tau$, then $\Omega' = \Omega = \LO(p,\RTs,\Lbls)$.
    \end{itemize}
\end{lemma}

\begin{proof}
    Keeping in mind \Cref{l:mpstMon:labelOracleDepOut}, $\Omega$ is generated with each possible $\alpha$ appearing exactly once under specific conditions and with unique continuation.
\end{proof}

\begin{theorem}[Transparency]
    \label{t:mpstMon:transparencyBisim}
    Suppose given
    \begin{itemize}
        \item a well-typed global type $\gtc{G_0}$,
        \item a participant $p \in \prt(\gtc{G_0})$, and
        \item a blackbox $P_0$.
    \end{itemize}
    Let
    \begin{itemize}
        \item $\RTs_0 \deq \braces{ ( q , \gtc{G_0} \wrt (p,q)^{\braces{\epsi}} ) \mid q \in \prt(\gtc{G_0}) \setminus \braces{p} }$, and
        \item $\mc{M_0} \deq \gtToMon( \gtc{G_0} , p , \prt(\gtc{G_0}) \setminus \braces{p} )$.
    \end{itemize}
    Suppose $\satisfies{ \monImpl{ \bufImpl{ p }{ P_0 }{ \epsi } }{ \mc{M_0} }{ \epsi } }{ \RTs_0 }[ p ]$ minimally (\Cref{d:mpstMon:minimalSat}).

    Let $\Omega_0 \deq \LO(p,\RTs_0,\emptyset)$.
    Then $\monImpl{ \bufImpl{ p }{ P_0 }{ \epsi } }{ \mc{M_0} }{ \epsi } \weakBisim{}{\Omega_0} \bufImpl{ p }{ P_0 }{ \epsi }$.
\end{theorem}

\begin{proof}
    By satisfaction (\Cref{d:mpstMon:satisfaction}), there exists a minimal satisfaction $\mcl{R}$ at $p$ such that $( \monImpl{ \bufImpl{ p }{ P_0 }{ \epsi } }{ \mc{M_0} }{ \epsi } , \RTs_0 , \emptyset ) \in \mcl{R}$.
    Let
    \[
        \mcl{B} \deq \{ \begin{array}[t]{@{}l@{}}
            ( \monImpl{ \bufImpl{ p }{ P }{ \vect m } }{ \mc{M} }{ \vect n } , \Omega , \bufImpl{ p }{ P }{ \vect n , \vect m } )
            \\
            {} \mid \forall P , M , \vect m , \vect n.~ \exists \RTs , \Lbls.~ \big( \begin{array}[t]{@{}l@{}}
                ( \monImpl{ \bufImpl{ p }{ P }{ \vect m } }{ \mc{M} }{ \vect n } , \RTs , \Lbls ) \in \mcl{R}
                \\ {} \wedge
                \coherent{G_0}{p}{\mc{M}}{\RTs}{\vect n}
                \\ {} \wedge
                \Omega = \LO( p , \RTs , \Lbls )
        \big) \}. \end{array} \end{array}
    \]
    Clearly, $( \monImpl{ \bufImpl{ p }{ P_0 }{ \epsi } }{ \mc{M_0} }{ \epsi } , \Omega_0 , \bufImpl{ p }{ P_0 }{ \epsi } ) \in \mcl{B}$, with $P = P_0 , \mc{M} = \mc{M_0} , \vect m = \vect n = \epsi$, \mbox{$\RTs = \RTs_0$}, $\Lbls = \emptyset$, and clearly $\coherent{G_0}{p}{\mc{M_0}}{\RTs_0}{\epsi}$.

    It remains to show that $\mcl{B}$ is a weak bisimulation.
    Take any $( \net P , \Omega , \net Q ) \in \mcl{B}$: there are $P , \mc{M} , \vect n , \vect m$ such that $\net P = \monImpl{ \bufImpl{ p }{ P }{ \vect m } }{ \mc{M} }{ \vect n }$, $\net Q = \bufImpl{ p }{ P }{ \vect n , \vect m }$, and there are $\RTs , \Lbls$ such that $( \net P , \RTs , \Lbls ) \in \mcl{R}$, $\coherent{\gtc{G_0}}{p}{\mc{M}}{\RTs}{\vect n}$, and $\Omega = \LO( p , \RTs , \Lbls )$.
    We show that the two conditions of \Cref{d:mpstMon:weakBisim} hold.

    \begin{enumerate}
        \item
            Take any $\net P' , \alpha , \Omega_1$ such that $\net P \bLtrans{\Omega}{ \alpha }{\Omega_1} \net P'$.
            The analysis depends on the rule from \Cref{d:mpstMon:enhancedLTS} used to derive the transition (Transition~\ruleLabel{buf}, \ruleLabel{dep}, or~\ruleLabel{no-dep}).
            We never need to read extra messages, so in each case we show the thesis for $\vect b \deq \epsi$ and $\net P'' \deq \net P'$.
            That is, in each case we show that there exists $\net Q'$ such $\net Q \bLtrans*{\Omega}{\alpha}{\Omega_1} \net Q'$ and $( \net P' , \Omega_1 , \net Q' ) \in \mcl{B}$.
            \begin{itemize}
                \item
                    Transition~\ruleLabel{buf-mon}.
                    Then $\alpha \in \braces{ p {?} q (x) , p {?} q \Parens{x} }$, $\net P' = \monImpl{ \bufImpl{ p }{ P }{ \vect m } }{ \mc{M} }{ u , \vect n }$ for $u \in \braces{ q {!} p (x) , q {!} p \Parens{x} }$, and $\Omega_1 = \Omega(\alpha)$.

                    By \Cref{l:mpstMon:LOAlpha}, we have $( q , \rtc{R^q} ) \in \RTs$ where $\unfold(\rtc{R^q})$ is a (dependency) message from $q$ to $p$.
                    For the sake of simplicity, assume w.l.o.g.\ that the message is no dependency.
                    Then $\unfold(\rtc{R^q}) = q {!} p^{\mbb L} \rtBraces{ \msg i<T_i> \rtc. \rtc{R^q_i} }_{i \in I}$, \mbox{$\alpha = p {?} q(x)$}, $x = \msg j<T_j>$ with~$j \in I$, and $u = q {!} p (x)$.
                    Let $\RTs' \deq \RTs \update{q \mapsto \rtc{R^q_j}}$ and \mbox{$\Lbls' \deq \Lbls \update{\mbb L \mapsto j}$}.
                    By \satref{sInput}{Input}, $( \net P' , \RTs' , \Lbls' ) \in \mcl{R}$.
                    By \Cref{l:mpstMon:LOAlpha}, $\Omega_1 = \LO( p , \RTs' , \Lbls' )$.

                    Since $\mc{M}$ is not updated, and the addition of $u$ to the buffer only adds an input transition to the coherent setup (\Cref{d:mpstMon:coherentSetup}) reflected by updated entry for~$q$ in $\RTs'$, we have $\coherent{ G_0 }{ p }{ \mc{M} }{ \RTs' }{ u , \vect n }$.
                    Let $\net Q' \deq \bufImpl{ p }{ P }{ u , \vect n , \vect m }$.
                    By Transition~\ruleLabel{buf}, $\net Q \bLtrans{\Omega}{ \alpha }{\Omega_1} \net Q'$ so $\net Q \bLtrans*{\Omega}{\alpha}{\Omega_1} \net Q'$.
                    Finally, by definition, \mbox{$( \net P' , \Omega_1 , \net Q' ) \in \mcl{B}$}.

                \item
                    Transition~\ruleLabel{buf-unmon}.
                    This rule does not apply to monitored blackboxes.

                \item
                    Transition~\ruleLabel{dep}.
                    Then $\alpha = \tau$, $\net P \ltrans{ \alpha' } \net P'$ with $\alpha' = s {!} r \Parens{j}$ for some $s,r,j$, and $\Omega_1 = \Omega$.

                    This can only have been derived from Transition~\ruleLabel{mon-out-dep}.
                    Then \mbox{$\mc{M} = p {!} (D \cup \braces{r}) (j) \mc. \mc{M'}$}, $s = p$, and \mbox{$\net P \ltrans{ \alpha' } \monImpl{ \bufImpl{ p }{ P }{ \vect m } }{ p {!} D (j) \mc. \mc{M'} }{ \vect n } = \net P'$}.
                    By \satref{sDependencyOutput}{Dependency output}, we have $( r , \rtc{R^r} ) \in \RTs$ with \mbox{$\unfold(\rtc{R^r}) = (p \lozenge q) {!} r^{\mbb L} \rtBraces{ i \rtc. \rtc{R^r_i} }_{i \in I}$} and $j \in I$, and so $( \net P' , \RTs' , \Lbls ) \in \mcl{R}$ with \mbox{$\RTs' \deq \RTs \update{r \mapsto \rtc{R^r_j}}$}.

                    In $\RTs'$, only the entry for $r$ has been updated, so \mbox{$\coherent{ \gtc{G_0} }{ p }{ p {!} D (j) \mc. \mc{M'} }{ \RTs' }{ \vect n }$}.
                    By \Cref{l:mpstMon:labelOracleDepOut}, $\Omega = \LO( p , \RTs' , \Lbls )$.
                    We have $\net Q \bLtrans*{\Omega}{\tau}{\Omega} \net Q$.
                    Finally, by definition, $( \net P' , \Omega , \net Q ) \in \mcl{B}$.

                \item
                    Transition~\ruleLabel{no-dep}.
                    Then $\alpha \neq s {!} r (\ell)$ for any $s,r,\ell$, $\net P \ltrans{ \alpha } \net P'$, and $\Omega_1 = \Omega(\alpha)$.
                    The analysis depends on the derivation of the transition.
                    Some rules are impossible: there is no parallel composition in $\net P$, no rules to derive transitions for buffered blackboxes are possible, and, by \Cref{l:mpstMon:satisfactionNoError}, no transitions resulting in an error signal are possible.

                    Some rules require to first unfold recursion in $\mc{M}$, derived by a number of consecutive applications of Transition~\ruleLabel{mon-rec}.
                    We apply induction on this number.
                    The inductive case is trivial by the IH.

                    In the base case, where there are no applications of Transition~\ruleLabel{mon-rec}, we consider each possible rule (Transition~\ruleLabel{mon-out}, \ruleLabel{mon-in}, \ruleLabel{mon-in-dep}, \ruleLabel{mon-tau}, \ruleLabel{mon-out-dep-empty}, and~\ruleLabel{mon-end}).

                    \begin{itemize}
                        \item
                            Transition~\ruleLabel{mon-out}.
                            Then $\alpha = p {!} q (\msg j<T_j>)$, $\bufImpl{ p }{ P }{ \vect m } \ltrans{ \alpha } \bufImpl{ p }{ P' }{ \vect m }$, \mbox{$\mc{M} = p {!} q \mBraces{ \msg i<T_i> \mc. \mc{M_i} }_{i \in I}$}, $j \in I$, and $\net P' = \monImpl{ \bufImpl{ p }{ P' }{ \vect m } }{ \mc{M_j} }{ \vect n }$.

                            The transition of the buffered blackbox must be due to Transition~\ruleLabel{buf-out}: $P \ltrans{ \alpha } P'$.
                            By \satref{sOutput}{Output}, we have $( q , \rtc{R^q} ) \in \RTs$ with $\unfold(\rtc{R^q}) = p {!} q^{\mbb L} \rtBraces{ \msg i<T_i> \rtc. \rtc{R^q_i} }_{i \in I}$ and $j \in I$, and so $( \net P' , \RTs' , \Lbls' ) \in \mcl{R}$ with $\RTs' \deq \RTs \update{q \mapsto \rtc{R^q_j}}$ and $\Lbls' \deq \Lbls \update{\mbb L \mapsto j}$.

                            In $\RTs'$, only the entry for $q$ has been updated, so \mbox{$\coherent{ \gtc{G_0} }{ p }{ M_j }{ \RTs' }{ \vect n }$}.
                            By \Cref{l:mpstMon:LOAlpha}, $\Omega_1 = \LO( p , \RTs' , \Lbls' )$.
                            Let $\net Q' \deq \bufImpl{ p }{ P' }{ \vect n , \vect m }$.
                            By Transition~\ruleLabel{buf-out}, $\net Q \ltrans{ \alpha } \net Q'$.
                            Then, by Transition~\ruleLabel{no-dep}, $\net Q \bLtrans{\Omega}{ \alpha }{\Omega_1} \net Q'$ so $\net Q \bLtrans*{\Omega}{ \alpha}{\Omega_1} \net Q'$.
                            Finally, by definition, $( \net P' , \Omega_1 , \net Q' ) \in \mcl{B}$.

                        \item
                            Transition~\ruleLabel{mon-in}.
                            Then $\alpha = \tau$, $\mc{M} = p {?} q \mBraces{ \msg i<T_i> \mc. \mc{M_i} }_{i \in I}$, $\vect n = \vect n' , u$ where $u = q {!} p (\msg j<T_j>)$, $j \in I$, $\net P' = \monImpl{ \bufImpl{ p }{ P }{ u , \vect m } }{ \mc{M_j} }{ \vect n' }$.

                            By \satref{sTau}{Tau}, $( \net P' , \RTs , \Lbls ) \in \mcl{R}$.
                            Since $\mc{M}$ has moved past the input from $q$ ($\RTs(q)$ was already past this point), and it has been removed from $\vect n$, we have $\coherent{ \gtc{G_0} }{ p }{ \mc{M_j} }{ \RTs }{ \vect n' }$.
                            By \Cref{l:mpstMon:LOAlpha}, $\Omega_1 = \LO( p , \RTs , \Lbls )$.
                            We have $\net Q \bLtrans*{\Omega}{\tau}{\Omega_1} \net Q$.
                            Finally, by definition, $( \net P' , \Omega_1 , \net Q ) \in \mcl{B}$.

                        \item
                            Transition~\ruleLabel{mon-in-dep}.
                            Analogous to Transition~\ruleLabel{mon-in}.

                        \item
                            Transition~\ruleLabel{mon-tau}.
                            Then $\alpha = \tau$, $\net P' = \monImpl{ \bufImpl{ p }{ P' }{ \vect m' } }{ \mc{M} }{ \vect n }$, and
                            \[
                                \bufImpl{ p }{ P }{ \vect m } \ltrans{ \tau } \bufImpl{ p }{ P' }{ \vect m' }.
                            \]

                            By \satref{sTau}{Tau}, $( \net P' , \RTs , \Lbls ) \in \mcl{R}$.
                            By \Cref{l:mpstMon:LOAlpha}, \mbox{$\Omega_1 = \LO( p , \RTs , \Lbls ) = \Omega$}.
                            Let $\net Q' \deq \bufImpl{ p }{ P' }{ \vect n , \vect m' }$.
                            The transition of the buffered blackbox is derived from Transition~\ruleLabel{buf-in}, \ruleLabel{buf-in-dep}, or~\ruleLabel{buf-tau}.
                            In any case, we can add messages to the back of the buffer without affecting the transition, such that $\net Q \ltrans{ \tau } \net Q'$.
                            By Transition~\ruleLabel{no-dep}, $\net Q \bLtrans{\Omega}{\tau}{\Omega_1} \net Q'$ so $\net Q \bLtrans*{\Omega}{\tau}{\Omega_1} \net Q'$.
                            Finally, by definition, $( \net P' , \Omega_1 , \net Q' ) \in \mcl{B}$.

                        \item
                            Transition~\ruleLabel{mon-out-dep-empty}.
                            Then $\alpha = \tau$, $\mc{M} = p {!} \emptyset (\ell) \mc. \mc{M'}$, and \mbox{$\net P' = \monImpl{ \bufImpl{ p }{ P }{ \vect m } }{ \mc{M'} }{ \vect n }$}.

                            By \satref{sTau}{Tau}, then $( \net P' , \RTs , \Lbls ) \in \mcl{R}$.
                            By \Cref{l:mpstMon:LOAlpha}, \mbox{$\Omega_1 = \LO( p , \RTs , \Lbls ) = \Omega$}.
                            Since the step from $\mc{M}$ to $\mc{M'}$ has no effect on $\RTs$, $\coherent{ \gtc{G_0} }{ p }{ \mc{M'} }{ \RTs }{ \vect n }$.
                            We have $\net Q \bLtrans*{\Omega}{\tau}{\Omega_1} \net Q$.
                            Then $( \net P' , \Omega_1 , \net Q ) \in \mcl{B}$.

                        \item
                            Transition~\ruleLabel{mon-end}.
                            Then $\alpha = \tEnd$, $\mc{M} = \tEnd$, $\vect n = \epsi$, $\bufImpl{ p }{ P }{ \vect m } \ltrans{\tEnd} \bufImpl{ p }{ P' }{ \vect m }$, and $\net P' = \monImpl{ \bufImpl{ p }{ P' }{ \vect m } }{ \checkmark }{ \epsi }$.

                            The transition of the buffered blackbox is derived from Transition~\ruleLabel{buf-end}: $P \ltrans{\tEnd} P'$.
                            Then by the same transition, $\net Q \ltrans{\tEnd} \bufImpl{ p }{ P' }{ \vect m } =: \net Q'$.
                            By \Cref{l:mpstMon:LOAlpha}, $\Omega_1 = \LO( p , \emptyset , \emptyset )$.
                            Then $\net Q \bLtrans*{\Omega}{\tEnd}{\Omega_1} \net Q'$.
                            By \satref{sEnd}{End}, for every $q \in \dom(\RTs)$, $\unfold\big(\RTs(q)\big) = \tEnd$, and $( \net P' , \emptyset , \emptyset ) \in \mcl{R}$.
                            Moreover, clearly, $\coherent{ G_0 }{ p }{ \checkmark }{ \emptyset }{ \epsi }$.
                            Then, by definition, $( \net P' , \Omega_1 , \net Q' ) \in \mcl{B}$.
                    \end{itemize}
            \end{itemize}

        \item
            Take any $\net Q',\alpha,\Omega_1$ such that $\net Q \bLtrans{\Omega}{ \alpha }{\Omega_1} \net Q' \checkmark$.
            The analysis depends on the rule from \Cref{d:mpstMon:enhancedLTS} used to derive the transition (Transition~\ruleLabel{buf}, \ruleLabel{dep}, or~\ruleLabel{no-dep}).
            \begin{itemize}
                \item
                    Transition~\ruleLabel{buf-mon}.
                    This rule does not apply to monitored blackboxes.

                \item
                    Transition~\ruleLabel{buf-unmon}.
                    Then $\alpha \in \braces{ p {?} q (x) , p {?} q \Parens{x} }$, $\net Q' = \bufImpl{ p }{ P }{ u , \vect n , \vect m }$ for $u \in \braces{ q {!} p (x) , q {!} p \Parens(x) }$, and $\Omega_1 = \Omega(\alpha)$.

                    By \Cref{l:mpstMon:LOAlpha}, we have $( q , \rtc{R^q} ) \in \RTs$ where $\unfold(\rtc{R^q})$ is a (dependency) message from $q$ to $p$.
                    For the sake of simplicity, assume w.l.o.g.\ that the message is no dependency.
                    Then $\unfold(\rtc{R^q}) = q {!} p^{\mbb L} \rtBraces{ \msg i<T_i> \rtc. \rtc{R^q_i} }_{i \in I}$, $\alpha = p {?} q (x)$, $x = \msg j<T_j>$ with $j \in I$, and $u = q {!} p (x)$.
                    Let $\RTs' \deq \RTs \update{q \mapsto \rtc{R^q_j}}$ and $\Lbls' \deq \Lbls \update{\mbb L \mapsto j}$.
                    Let $\net P' \deq \monImpl{ \bufImpl{ p }{ P }{ \vect m } }{ \mc{M} }{ u , \vect n }$.
                    By \satref{sInput}{Input}, $( \net P' , \RTs' , \Lbls' ) \in \mcl{R}$.

                    Since $\mc{M}$ is not updated, and the addition of $u$ to the buffer only adds an input to the coherent setup (\Cref{d:mpstMon:coherentSetup}) reflected by the updated for $q$ in $\RTs'$, we have $\coherent{ \gtc{G_0} }{ p }{ \mc{M} }{ \RTs' }{ u , \vect n }$.
                    By \Cref{l:mpstMon:LOAlpha}, $\Omega_1 = \LO( p , \RTs' , \Lbls' )$.
                    Let $\Omega_2 \deq \Omega_1$ and $\vect b \deq \epsi$.
                    We have $\net P \bLtrans*{\Omega}{\vect b}{\Omega} \net P$.
                    By Transition~\ruleLabel{buf}, $\net P \bLtrans{\Omega}{\alpha}{\Omega_2} \net P'$ so $\net P \bLtrans*{\Omega}{\vect b , \alpha}{\Omega_2} \net P'$.
                    Let $\net Q'' \deq \net Q'$; we have $\net Q' \bLtrans*{\Omega_2}{\vect b}{\Omega_2} \net Q''$.
                    Then, since $\net Q \bLtrans{\Omega}{\alpha}{\Omega_2} \net Q'$, we have $\net Q \bLtrans{\Omega}{\alpha , \vect b}{\Omega_2} \net Q''$.
                    Finally, by definition, $( \net P' , \Omega_2 , \net Q'' ) \in \mcl{B}$.

                \item
                    Transition~\ruleLabel{dep}.
                    Then $\alpha = \tau$, $\net Q \ltrans{\alpha'} \net Q'$ with $\alpha' = s {!} r (\ell)$ for some $s,r,\ell$.
                    There are no rules to derive this transition, so this case does not apply.

                \item
                    Transition~\ruleLabel{no-dep}.
                    Then $\alpha \neq s {!} r (\ell)$ for any $s,r,\ell$, $\net Q \ltrans{\alpha} \net Q'$, and $\Omega_1 = \Omega(\alpha)$.
                    The analysis depends on the derivation of the transition.
                    Some rules are impossible: there is no parallel composition in $\net Q$, and no rules to derive transitions for monitored blackboxes are possible.

                    As in case~1 above, we may first need to unfold recursion in $\mc{M}$, which we do inductively.
                    We consider each possible rule (Transition~\ruleLabel{buf-out}, \ruleLabel{buf-in}, \ruleLabel{buf-in-dep}, \ruleLabel{buf-tau}, and~\ruleLabel{buf-end}).
                    \begin{itemize}
                        \item
                            Transition~\ruleLabel{buf-out}.
                            Then $\alpha = p {!} q (\msg j<T_j>)$, $\net Q' = \bufImpl{ p }{ P' }{ \vect n , \vect m }$, and $P \ltrans{\alpha} P'$.
                            Then, also by Transition~\ruleLabel{buf-out}, $\bufImpl{ p }{ P }{ \vect m } \ltrans{\alpha} \bufImpl{ p }{ P' }{ \vect m }$.

                            By \Cref{l:mpstMon:satisfactionOutput}, $\mc{M} = \mc{M'}$ or $\mc{M} = p {!} D (\ell) \mc. \mc{M'}$ where $\mc{M'} = p {!} q \mBraces{ \msg i<T_i> \mc. \mc{M'_i} }_{i \in I}$ for some $I \supseteq \braces{j}$.
                            W.l.o.g., assume the latter.
                            Let $\net P' \deq \monImpl{ \bufImpl{ p }{ P }{ \vect m } }{ \mc{M'} }{ \vect n }$.
                            By induction on the size of $D = \braces{ r_1 , \ldots , r_k }$, we show that $\net P \bLtrans*{\Omega}{\tau }{\Omega} \net P'$, where there exists $\RTs'$ such that $( \net P' , \RTs' , \Lbls ) \in \mcl{R}$, $\coherent{ \gtc{G_0} }{ p }{ \mc{M'} }{ \RTs' }{ \vect n }$, and $\Omega = \LO( p , \RTs' , \Lbls )$.

                            In the base case, $D = \emptyset$.
                            By Transition~\ruleLabel{mon-out-dep-empty}, $\net P \ltrans{ \tau } \net P'$.
                            Then, by Transition~\ruleLabel{no-dep}, $\net P \bLtrans{\Omega}{\tau}{\Omega} \net P'$ so $\net P \bLtrans*{\Omega}{\tau}{\Omega} \net P'$.
                            Let $\RTs'\deq\RTs$.
                            By \satref{sTau}{Tau}, $( \net P' , \RTs' , \Lbls ) \in \mcl{R}$.
                            Since the step from $\mc{M}$ to $\mc{M'}$ does not affect any relative types, then also $\coherent{ \gtc{G_0} }{ p }{ \mc{M'} }{ \RTs' }{ \vect n }$.
                            The condition on $\Omega$ holds by \Cref{l:mpstMon:LOAlpha}.

                            In the inductive case, $D = D' \cup \braces{r_k}$.
                            By Transition~\ruleLabel{mon-out-dep},
                            \[
                                \net P \ltrans{ p {!} r_k \Parens{\ell} } \monImpl{ \bufImpl{ p }{ P }{ \vect m } }{ p {!} D' (\ell) \mc. \mc{M'} }{ \vect n }.
                            \]
                            Then, by Transition~\ruleLabel{dep}, \mbox{$\net P \bLtrans{\Omega}{\tau}{\Omega} \monImpl{ \bufImpl{ p }{ P }{ \vect m } }{ p {!} D' (\ell) \mc. \mc{M'} }{ \vect n }$}.
                            By \satref{sDependencyOutput}{Dependency output}, there exists $\RTs'$ such that \mbox{$( \monImpl{ \bufImpl{ p }{ P }{ \vect m } }{ p {!} D' (\ell) \mc. \mc{M'} }{ \vect n } , \RTs' , \Lbls ) \in \mcl{R}$}.
                            Since $\RTs'$ only updates the entry of $r$ (which was a dependency output in $\RTs$), also \mbox{$\coherent{ \gtc{G_0} }{ p }{ p {!} D' (\ell) \mc. \mc{M'} }{ \RTs' , \vect n }$}.
                            Using the same reasoning, by \Cref{l:mpstMon:labelOracleDepOut}, $\Omega = \LO( p , \RTs' , \Lbls )$.
                            Then, by the IH, $\monImpl{ \bufImpl{ p }{ P }{ \vect m } }{ p {!} D' (\ell) \mc. \mc{M'} }{ \vect n } \bLtrans*{\Omega}{\tau}{\Omega} \net P'$ so $\net P \bLtrans*{\Omega}{\tau}{\Omega} \net P'$.

                            Now, let $\net P'' \deq \monImpl{ \bufImpl{ p }{ P' }{ \vect m } }{ \mc{M'_j} }{ \vect n }$.
                            By Transition~\ruleLabel{mon-out}, $\net P' \ltrans{\alpha} \net P''$.
                            Let $\vect b \deq \epsi$ and $\Omega_2 \deq \Omega_1$.
                            By Transition~\ruleLabel{no-dep}, since $\Omega_2 = \Omega(\alpha)$, $\net P' \bLtrans{\Omega}{\alpha}{\Omega_2} \net P''$ so $\net P \bLtrans*{\Omega}{\vect b , \alpha}{\Omega_2} \net P''$.

                            By \satref{sOutput}{Output}, we have $( q , \rtc{R^q} ) \in \RTs'$ with
                            \[
                                \unfold(\rtc{R^q}) = p {!} q^{\mbb{L}} \rtBraces{ \msg i<T_i> \rtc. \rtc{R^q_i} }_{i \in I}
                            \]
                            and $j \in I$.
                            Let $\RTs'' \deq \RTs' \update{q \mapsto \rtc{R^q_j}}$ and \mbox{$\Lbls' \deq \Lbls \update{\mbb L \mapsto j}$}.
                            Then $( \net P'' , \RTs'' , \Lbls' ) \in \mcl{R}$.
                            Since $\RTs''$ only updates the entry for $r$ (which was an output in $\RTs'$), we have $\coherent{ \gtc{G_0} }{ p }{ \mc{M'_j} }{ \RTs'' }{ \vect n }$.
                            By \Cref{l:mpstMon:LOAlpha}, $\Omega_2 = \LO( p , \RTs'' , \Lbls' )$.
                            Trivially, $\net Q' \bLtrans*{\Omega_1}{\vect b}{\Omega_2} \net Q'$, so $\net Q \bLtrans*{\Omega}{ \alpha , \vect b }{\Omega_2} \net Q'$.
                            Finally, by definition, $( \net P'' , \Omega_2 , \net Q' ) \in \mcl{B}$.

                        \item
                            Transition~\ruleLabel{buf-in}.
                            Then $\alpha = \tau$, $\vect n , \vect m = \vect n' , \vect m' , u$ where \mbox{$u = q {!} p (\msg j<T_j>)$}, $\net Q' = \bufImpl{ p }{ P' }{ \vect n' , \vect m' }$, and $P \ltrans{ p {?} q (\msg j<T_j>) } P'$.
                            By \Cref{l:mpstMon:LOAlpha}, \mbox{$\Omega_1 = \Omega = \LO( p , \RTs , \Lbls )$}.
                            We know $u$ appears in $\vect n$ or $\vect m$.
                            We discuss each case separately.
                            \begin{itemize}
                                \item
                                    We have $u$ appears in $\vect n$.
                                    Then $\vect n = \vect n' , u$ and $\vect m = \vect m'$, and there are no messages in $\vect m$ with sender $q$.

                                    By \Cref{l:mpstMon:finiteReceive}, from $\mc{M}$ any path leads to the input by $p$ from $q$ in finitely many steps.
                                    We show by induction on the maximal number of such steps that there are $\vect b , \vect c , \vect d , \Omega_2$ such that
                                    \[
                                        \net P \bLtrans*{\Omega}{\vect b}{\Omega_2} \monImpl{ \bufImpl{ p }{ P }{ \vect d , \vect m } }{ p {?} q \mBraces{ \msg i<T_i> \mc. \mc{M'_i} }_{i \in I} }{ \vect c , u } =: \net P'
                                    \]
                                    where $j \in I$, $\net Q' \bLtrans*{\Omega_1}{\vect b}{\Omega_2} \bufImpl{ p }{ P' }{ \vect c , \vect d , \vect m } =: \net Q''$, and that there exist $\RTs',\Lbls'$ such that $( \net P' , \RTs' , \Lbls' ) \in \mcl{R}$, \mbox{$\coherent{ \gtc{G_0} }{ p }{ p {?} q \mBraces{ \msg i<T_i> \mc. \mc{M'_i} }_{i \in I} }{ \RTs' }{ \vect c , u }$}, and $\Omega_2 = \LO( p , \RTs' , \Lbls' )$.

                                    In the base case, $\mc{M} = p {?} q \mBraces{ \msg i<T_i> \mc. \mc{M'_i} }_{i \in I}$ with $j \in I$.
                                    Let $\vect b \deq \vect d \deq \epsi$, \mbox{$\vect c \deq \vect n'$}, $\Omega_2 \deq \Omega$.
                                    Then $\net P' = \net P$ and $\net Q'' = \net Q'$, so the thesis holds trivially.

                                    In the inductive case, by \Cref{l:mpstMon:monitorNoOutput},
                                    \[
                                        \mc{M} \in \braces{ p {?} r \mBraces{ \msg k<T_k> \mc. \mc{M_k} }_{k \in K} , p {?} r \mBraces{ k \mc. \mc{M_k} }_{k \in K} , p {!} D (\ell) \mc. \mc{M'} \mid r \neq q }.
                                    \]
                                    We discuss each case separately.
                                    \begin{itemize}
                                        \item
                                            $\mc{M} = p {?} r \mBraces{ \msg k<T_k> \mc. \mc{M_k} }_{k \in K}$ for $r \neq q$.
                                            This case depends on whether there is a message from $r$ in $\vect n'$.
                                            We discuss each case separately.

                                            \medskip
                                            If there is a message from $r$ in $\vect n'$, then $\vect n' , u = \vect n'' , u , w$ where \mbox{$w = r {!} p (\msg k'<T_{k'}>)$} for $k' \in K$.
                                            Let \mbox{$\net P'' \deq \monImpl{ \bufImpl{ p }{ P }{ w , \vect m } }{ \mc{M_{k'}} }{ \vect n'' , u }$}.
                                            By Transition~\ruleLabel{mon-in}, $\net P \ltrans{\tau} \net P''$.
                                            By \satref{sTau}{Tau}, \mbox{$( \net P'' , \RTs , \Lbls ) \in \mcl{R}$}.
                                            Since both $\vect n'$ and $\mc{M}$ have correspondingly updated, \mbox{$\coherent{ \gtc{G_0} }{ p }{ \mc{M_{k'}} }{ \RTs }{ \vect n'' , u }$}.
                                            By \Cref{l:mpstMon:LOAlpha}, we have \mbox{$\Omega(\tau) = \Omega$}.
                                            By Transition~\ruleLabel{no-dep}, $\net P \bLtrans{\Omega}{\tau}{\Omega} \net P''$.
                                            By the IH, there are $\vect b , \vect c , \vect d , \Omega_2$ such that $\net P'' \bLtrans*{\Omega}{\vect b}{\Omega_2} \net P'$ so $\net P \bLtrans*{\Omega}{\vect b}{\Omega_2} \net P'$, and \mbox{$\net Q' \bLtrans*{\Omega_1}{\vect b}{\Omega_2} \net Q''$}
                                            Moreover, there are $\RTs',\Lbls'$ such that $( \net P' , \RTs' , \Lbls' ) \in \mcl{R}$, \mbox{$\coherent{ \gtc{G_0} }{ p }{ p {?} q \mBraces{ \msg i<T_i> \mc. \mc{M'_i} }_{i \in I} }{ \RTs' }{ \vect c , u }$}, and $\Omega_2 = \LO( p , \RTs' , \Lbls' )$.

                                            \medskip
                                            If there is no message from $r$ in $\vect n'$, by the definition of coherent setup (\Cref{d:mpstMon:coherentSetup}), we have $( r , \rtc{R^r} ) \in \RTs$ with \mbox{$\unfold(\rtc{R^r}) = r {!} p^{\mbb L} \rtBraces{ \msg k<T_k> \rtc. \rtc{R^r_k} }_{k \in K}$}.
                                            Take any $k' \in K$, and let \mbox{$w \deq r {!} p (\msg k'<T_{k'}>)$}.
                                            Let $\RTs'' \deq \RTs \update{r \mapsto \rtc{R^r_{k'}}}$ and \mbox{$\Lbls'' \deq \Lbls \update{\mbb L \mapsto k'}$}.
                                            Also, let \mbox{$\Omega'_2 \deq \LO( p , \RTs'' , \Lbls'' )$}.
                                            Then, by \Cref{l:mpstMon:LOAlpha}, $\Omega(p {?} r (\msg k'<T_{k'}>)) = \Omega'_2$.

                                            Let $\net P''_1 \deq \monImpl{ \bufImpl{ p }{ P }{ \vect m } }{ \mc{M} }{ \vect n' , u , w }$ and $\net Q''' \deq \bufImpl{ p }{ P' }{ \vect n' , w , \vect m }$.
                                            By \satref{sInput}{Input}, $(\net P''_1 , \RTs'' , \Lbls'' ) \in \mcl{R}$.
                                            Since the buffer and relative types have changed accordingly, \mbox{$\coherent{ \gtc{G_0} }{ p }{ \mc{M} }{ \RTs'' }{ \vect n' , u , w }$}.
                                            By Transition~\ruleLabel{buf}, $\net P \bLtrans{\Omega}{w}{\Omega'_2} \net P''_1$ and $\net Q' \bLtrans{\Omega_1}{w}{\Omega'_2} \net Q'''$.
                                            Let \mbox{$\net P''_2 \deq \monImpl{ \bufImpl{ p }{ P }{ w , \vect m } }{ \mc{M_{k'}} }{ \vect n' , u }$}.
                                            By Transition~\ruleLabel{mon-in}, $\net P''_1 \ltrans{\tau} \net P''_2$.
                                            By \Cref{l:mpstMon:LOAlpha}, $\Omega'_2(\tau) = \Omega'_2$.
                                            Then, by Transition~\ruleLabel{no-dep}, \mbox{$\net P''_1 \bLtrans{\Omega'_2}{\tau}{\Omega'_2} \net P''_2$}.
                                            By \satref{sTau}{Tau}, $(\net P''_2 , \RTs'' , \Lbls'' ) \in \mcl{R}$.
                                            Since the buffer and monitor have changed accordingly,
                                            \[
                                                \coherent{ \gtc{G_0} }{ p }{ \mc{M_{k'}} }{ \RTs'' }{ \vect n' , u }.
                                            \]

                                            By the IH, there are $\vect b , \vect c , \vect d , \Omega_2$ such that $\net P''_2 \bLtrans*{\Omega'_2}{\vect b}{\Omega_2} \net P'$ so $\net P \bLtrans*{\Omega}{w,\vect b}{\Omega_2} \net P'$, and $\net Q''' \bLtrans*{\Omega'_2}{\vect b}{\Omega_2} \net Q''$ so $\net Q' \bLtrans*{\Omega_1}{w,\vect b}{\Omega_2} \net Q''$.
                                            Moreover, there are $\RTs',\Lbls'$ such that $( \net P' , \RTs' , \Lbls' ) \in \mcl{R}$, \mbox{$\coherent{ \gtc{G_0} }{ p }{ p {?} q \mBraces{ \msg i<T_i> \mc. \mc{M'_i} }_{i \in I} }{ \RTs' }{ \vect c , u }$}, and $\Omega_2 = \LO( p , \RTs' , \Lbls' )$.

                                        \item
                                            $\mc{M} = p {?} r \mBraces{ k \mc. \mc{M_k} }_{k \in K}$ for $r \neq q$.
                                            This case is analogous to the one above.

                                        \item
                                            $\mc{M} = p {!} D (\ell) \mc. \mc{M'}$.
                                            Let $\net P'' \deq \monImpl{ \bufImpl{ p }{ P }{ \vect m } }{ \mc{M'} }{ \vect n' , u }$.
                                            Similar to the case of Transition~\ruleLabel{buf-out} above, $\net P \bLtrans*{\Omega}{\tau}{\Omega} \net P''$.
                                            Moreover, there are $\RTs'',\Lbls''$ such that \mbox{$(\net P'' , \RTs'' , \Lbls'' ) \in \mcl{R}$}, and \mbox{$\coherent{ \gtc{G_0} }{ p }{ \mc{M'} }{ \RTs'' }{ \vect n' , u }$}.
                                            By \Cref{l:mpstMon:labelOracleDepOut}, \mbox{$\Omega = \LO( p , \RTs'' , \Lbls'' )$}.

                                            By the IH, there are $\vect b , \vect c , \vect d , \Omega_2$ such that $\net P'' \bLtrans*{\Omega}{\vect b}{\Omega_2} \net P'$ so \mbox{$\net P \bLtrans*{\Omega}{\vect b}{\Omega_3} \net P'$}, and $\net Q' \bLtrans*{\Omega_1}{\vect b}{\Omega_2} \net Q''$
                                            Moreover, there are $\RTs',\Lbls'$ such that $( \net P' , \RTs' , \Lbls' ) \in \mcl{R}$, $\coherent{ \gtc{G_0} }{ p }{ p {?} q \mBraces{ \msg i<T_i> \mc. \mc{M'_i} }_{i \in I} }{ \RTs' }{ \vect c , u }$, and $\Omega_2 = \LO( p , \RTs' , \Lbls' )$.
                                    \end{itemize}

                                    Let $\net P'''_1 \deq \monImpl{ \bufImpl{ p }{ P }{ \vect d , \vect m , u } }{ \mc{M'_j} }{ \vect c }$.
                                    By Transition~\ruleLabel{mon-in}, $\net P' \ltrans{\tau} \net P'''_1$.
                                    Since the buffer and monitor changed accordingly, $\coherent{ \gtc{G_0} }{ p }{ \mc{M'_j} }{ \RTs' }{ \vect c }$.
                                    By Transition~\ruleLabel{no-dep}, $\net P' \bLtrans{\Omega_2}{\tau}{\Omega_2} \net P'''_1$.
                                    Let $\net P'''_2 \deq \monImpl{ \bufImpl{ p }{ P' }{ \vect d , \vect m } }{ \mc{M'_j} }{ \vect c }$.
                                    By \Cref{l:mpstMon:LOAlpha}, $\Omega_2(\tau) = \Omega_2$.
                                    By Transition~\ruleLabel{buf-in},
                                    \[
                                        \bufImpl{ p }{ P }{ \vect d , \vect m , u } \ltrans{\tau} \bufImpl{ p }{ P' }{ \vect d , \vect m },
                                    \]
                                    so, by Transition~\ruleLabel{mon-tau}, $\net P'''_1 \ltrans{\tau} \net P'''_2$.
                                    By Transition~\ruleLabel{no-dep}, $\net P'''_1 \bLtrans{\Omega_2}{\tau}{\Omega_2} \net P'''_2$.
                                    By \satref{sTau}{Tau}, $(\net P'''_2 , \RTs' , \Lbls') \in \mcl{R}$.
                                    We have $\net P \bLtrans*{\Omega}{\vect b,\tau}{\Omega_2} \net P'''_2$ and $\net Q \bLtrans{\Omega}{\tau}{\Omega_1} \net Q' \bLtrans*{\Omega_1}{\vect b}{\Omega_2} \net Q''$.
                                    Recall that \mbox{$\Omega_2 = \LO( p , \RTs' , \Lbls' )$}.
                                    Then, by definition, $( \net P'''_2 , \Omega_2 , \net Q'' ) \in \mcl{B}$.

                                \item
                                    We have $u$ appears in $\vect m$.
                                    Then $\vect n = \vect n'$ and $\vect m = \vect m' , u$.
                                    Then, by Transition~\ruleLabel{buf-in}, $\bufImpl{ p }{ P }{ \vect m' , u } \ltrans{ \tau } \bufImpl{ p }{ P' }{ \vect m' }$.

                                    Let $\net P' \deq \monImpl{ \bufImpl{ p }{ P' }{ \vect m' } }{ \mc{M} }{ \vect n }$.
                                    By Transition~\ruleLabel{mon-tau}, $\net P \ltrans{\tau} \net P'$.

                                    By \Cref{l:mpstMon:LOAlpha}, $\Omega(\tau) = \Omega$, so $\Omega_1 = \Omega$.
                                    Let $\Omega_2 \deq \Omega$.
                                    Then, by Transition~\ruleLabel{no-dep}, $\net P \bLtrans{\Omega}{\tau}{\Omega_2} \net P'$ so $\net P \bLtrans*{\Omega}{\tau}{\Omega_2} \net P'$.
                                    By \satref{sTau}{Tau}, $( \net P' , \RTs , \Lbls ) \in \mcl{R}$.
                                    Finally, $\net Q \bLtrans{\Omega}{\tau}{\Omega_1} \net Q' \bLtrans*{\Omega_1}{\epsi}{\Omega_2} \net Q'$.
                                    Then \mbox{$( \net P' , \Omega_2 , \net Q' ) \in \mcl{B}$}.
                            \end{itemize}

                        \item
                            Transition~\ruleLabel{buf-in-dep}.
                            Analogous to Transition~\ruleLabel{buf-in}.

                        \item
                            Transition~\ruleLabel{buf-tau}.
                            Then $\alpha = \tau$, $\net Q' = \bufImpl{ p }{ P' }{ \vect n , \vect m }$, and $P \ltrans{ \tau } P'$.
                            By Transition~\ruleLabel{buf-tau}, also $\bufImpl{ p }{ P }{ \vect m } \ltrans{ \tau } \bufImpl{ p }{ P' }{ \vect m }$.

                            By \Cref{l:mpstMon:monitorNotCheck}, $\mc{M} \neq \checkmark$.
                            Let $\net P' \deq \monImpl{ \bufImpl{ p }{ P' }{ \vect m } }{ \mc{M} }{ \vect n }$.
                            By Transition~\ruleLabel{mon-tau}, $\net P \ltrans{ \tau } \net P'$.

                            By \Cref{l:mpstMon:LOAlpha}, $\Omega(\tau) = \Omega$, so $\Omega_1 = \Omega$.
                            Let $\Omega_2 \deq \Omega$.
                            Then, by Transition~\ruleLabel{no-dep}, $\net P \bLtrans{\Omega}{ \alpha }{\Omega} \net P'$ so $\net P \bLtrans*{\Omega}{\epsi , \alpha }{\Omega_2} \net P'$.
                            By \satref{sTau}{Tau}, $( \net P' , \RTs , \Lbls ) \in \mcl{R}$.
                            Finally, we have $\net Q \bLtrans{\Omega}{\alpha}{\Omega_1} \net Q' \bLtrans*{\Omega_1}{\epsi}{\Omega_2} \net Q'$.
                            Then $( \net P' , \Omega_2 , \net Q' ) \in \mcl{B}$.

                        \item
                            Transition~\ruleLabel{buf-end}.
                            Then $\alpha = \tEnd$, $P \ltrans{ \tEnd } P'$, and $\net Q' = \bufImpl{ p }{ P' }{ \vect n , \vect m }$.
                            By Transition~\ruleLabel{buf-end}, also $\bufImpl{ p }{ P }{ \vect m } \ltrans{\tEnd} \bufImpl{ p }{ P' }{ \vect m }$.

                            By \Cref{l:mpstMon:satisfactionEnd}, $\vect n = \epsi$, and $\mc{M} = \tEnd$ or $\mc{M} = p {!} D (\ell) \mc. \tEnd$. W.l.o.g., assume the latter.
                            Let $\net P' \deq \monImpl{ \bufImpl{ p }{ P }{ \vect m } }{ \tEnd }{ \epsi }$.
                            Similar to the case for Transition~\ruleLabel{buf-out}, $\net P \bLtrans*{\Omega}{\tau}{\Omega} \net P'$.

                            Let $\net P'' \deq \monImpl{ \bufImpl{ p }{ P' }{ \vect m } }{ \checkmark }{ \epsi }$.
                            By Transition~\ruleLabel{mon-end}, $\net P' \ltrans{\tEnd} \net P''$.
                            Let $\vect b \deq \epsi$ and $\Omega_2 \deq \LO( p , \emptyset , \emptyset )$.
                            By \Cref{l:mpstMon:LOAlpha}, $\Omega(\tEnd) = \Omega_2 = \Omega_1$.
                            By Transition~\ruleLabel{no-dep}, since $( \tEnd , \Omega_2 ) \in \Omega$, $\net P' \bLtrans{\Omega}{\tEnd}{\Omega_2} \net P''$, so $\net P \bLtrans*{\Omega}{\vect b , \tEnd}{\Omega_2} \net P''$.

                            By \satref{sEnd}{End}, for every $q \in \dom(\RTs)$, $\unfold\big(\RTs(q)\big) = \tEnd$, and $( \net P'' , \emptyset , \emptyset ) \in \mcl{R}$.
                            Clearly, $\coherent{\gtc{G_0}}{p}{\checkmark}{\emptyset}{\vect n}$.
                            Then, trivially, $\net Q' \bLtrans*{\Omega_1}{\vect b}{\Omega_2} \net Q'$, so $\net Q \bLtrans*{\Omega}{\tEnd , \vect b}{\Omega_2} \net Q'$.
                            Finally, by definition, $( \net P'' , \Omega_2 , \net Q' ) \in \mcl{B}$.
                            \qedhere
                    \end{itemize}
            \end{itemize}
    \end{enumerate}
\end{proof}